\documentclass[a4paper,11pt]{article}
\usepackage[english]{babel}
\usepackage{amsmath,amssymb,graphicx,stmaryrd,mathrsfs,xcolor,latexsym,theorem}
\usepackage{xcolor}
\usepackage[left= 2.5cm, bottom = 4cm, top = 3.5cm, right= 2.5cm]{geometry}
\usepackage[citecolor=blue,colorlinks=true]{hyperref}
\usepackage{authblk} 
\usepackage{tikz}

\usepackage[T1]{fontenc} 

\newcommand{\assign}{:=}
\newcommand{\backassign}{=:}
\newcommand{\cdummy}{\cdot}
\newcommand{\mathD}{\mathrm{D}}
\newcommand{\mathd}{\mathrm{d}}
\newcommand{\precprec}{\prec\!\!\!\prec}

\newcommand{\tmcolor}[2]{{\color{#1}{#2}}}

\newcommand{\tmmathbf}[1]{\ensuremath{\boldsymbol{#1}}}
\newcommand{\tmop}[1]{\ensuremath{\operatorname{#1}}}

\newcommand{\tmtextit}[1]{{\itshape{#1}}}
\newenvironment{proof}{\noindent\textbf{Proof\ }}{\hspace*{\fill}$\Box$\medskip}
\newtheorem{theorem}{Theorem}[section]
\newtheorem{lemma}[theorem]{Lemma}
\newtheorem{proposition}[theorem]{Proposition}
\newtheorem{corollary}[theorem]{Corollary}
{\theorembodyfont{\rmfamily}\newtheorem{remark}[theorem]{Remark}}

\numberwithin{equation}{section}
\newcommand{\tmkeywords}{\textbf{Keywords:} }

\newcommand{\CC}{\mathscr{C} \hspace{.1em}}

\newcommand{\LL}{\mathscr{L} \hspace{.2em}}
\newcommand{\Q}{\mathscr{Q} \hspace{.2em}}

\newcommand{\UU}{\mathscr{U}}

\begin{document}

\title{A PDE construction of \\ the Euclidean $\Phi^4_3$ quantum field theory}

\author[1]{Massimiliano Gubinelli}
\author[2]{Martina Hofmanov\'a}
\affil[1]{\small Hausdorff Center for Mathematics\\ \& Institute for Applied Mathematics\\ University of Bonn\\
Endenicher Allee 60\\
53115 Bonn, Germany.  \href{mailto:gubinelli@iam.uni-bonn.de}{gubinelli@iam.uni-bonn.de} }
%
\affil[2]{\small Fakult\"at f\"ur Mathematik, Universit\"at Bielefeld, Postfach 10 01 31, 33501 Bielefeld, Germany. \href{mailto:hofmanova@math.uni-bielefeld.de}{hofmanova@math.uni-bielefeld.de}}

\maketitle

\begin{abstract}
  We present a new construction of the Euclidean $\Phi^4$ quantum
  field theory on $\mathbb{R}^3$ based on PDE arguments. More precisely, we
  consider an approximation of the stochastic quantization equation on
  $\mathbb{R}^3$ defined on a periodic lattice of mesh size $\varepsilon$ and
  side length $M$. We introduce 
  a new renormalized energy method in weighted spaces and prove tightness of the corresponding Gibbs measures as
  $\varepsilon \rightarrow 0$, $M \rightarrow \infty$. Every limit point is non-Gaussian and satisfies reflection positivity, translation invariance and stretched exponential integrability. These properties allow to verify the Osterwalder--Schrader axioms for a  Euclidean QFT apart from rotation invariance and clustering. Our argument applies to arbitrary positive coupling constant, to multicomponent models with $O(N)$ symmetry and to some long-range variants. Moreover, we establish an integration by parts formula leading to the hierarchy of Dyson--Schwinger equations for the Euclidean correlation functions. To this end, we identify the renormalized cubic term as a \emph{distribution} on the space of Euclidean fields.
\end{abstract}

\tmkeywords{stochastic quantization, Euclidean quantum field theory,
paracontrolled calculus, integration by parts formulas, Dyson--Schwinger equations}

{\tableofcontents}

\section{Introduction}\label{sec:intro}

Let 
$\Lambda_{M, \varepsilon} = ( (\varepsilon\mathbb{Z})/( M\mathbb{Z}))^3$ be a   periodic lattice with
mesh size $\varepsilon$ and side length $M$ where $M/(2\varepsilon)\in\mathbb{N} .$ Consider the family $(\nu_{M, \varepsilon})_{M,\varepsilon}$ 
of Gibbs measures for the scalar field $\varphi:\Lambda_{M, \varepsilon}\to \mathbb{R}$, given by
\begin{equation}
 \mathd \nu_{M, \varepsilon} \propto \exp \left\{ - 2 \varepsilon^d  \sum_{x \in\Lambda_{M, \varepsilon}} \left[
  \frac{\lambda}{4} | \varphi_x |^4 + \frac{- 3 \lambda a_{M, \varepsilon} + 3
  \lambda^2 b_{M, \varepsilon} + m^2}{2} | \varphi_x |^2 + \frac{1}{2} |
  \nabla_{\varepsilon} \varphi_x |^2 \right] \right\}  \prod_{x \in \Lambda_{M,
  \varepsilon}}\!\! \mathd \varphi_x, \label{eq:gibbs}
\end{equation}
where $\nabla_{\varepsilon}$ denotes the discrete gradient and $a_{M,
\varepsilon}, b_{M, \varepsilon}$ are suitable renormalization constants, $m^2
\in \mathbb{R}$ is called the \tmtextit{mass} and $\lambda > 0$ the
\tmtextit{coupling constant}.  The numerical factor in the exponential is chosen in order to simplify the form of the  stochastic quantization equation \eqref{eq:P4} below. 
The main result of this paper is the following.
\begin{theorem}
  \label{th:main}There exists a choice of the sequence $(a_{M, \varepsilon},
  b_{M, \varepsilon})_{M, \varepsilon}$ such that for any $\lambda > 0$ and
  $m^2 \in \mathbb{R}$, the family of measures $(\nu_{M, \varepsilon})_{M,
  \varepsilon}$ appropriately extended to $\mathcal{S}' (\mathbb{R}^3)$ is tight.
  Every accumulation point $\nu$ is translation invariant, reflection positive
  and non-Gaussian. In addition, for every small $\kappa > 0$ there exists $\sigma >
  0$, $\beta > 0$ and $\upsilon = O (\kappa) > 0$ such that
  \begin{equation}
    \int_{\mathcal{S}' (\mathbb{R}^3)} \exp\{{\beta \| (1 + | \cdot |^2)^{-
    \sigma} \varphi \|_{H^{- 1 / 2 - \kappa}}^{1 - \upsilon}}\} \nu (\mathd
    \varphi) < \infty . \label{eq:exp-int-intro}
  \end{equation}
  Every $\nu$ satisfies an integration by parts formula which leads
  to the hierarchy of the~Dyson--Schwinger equations for $n$-point correlation
  functions.
\end{theorem}

For the precise definition of translation invariance and reflection positivity (RP) we refer the reader to Section~\ref{s:ax}.

\medskip

The proof of convergence of the family $(\nu_{M, \varepsilon})_{M,  \varepsilon}$ 
has been one of the major achievements of the constructive quantum field theory (CQFT)
program~{\cite{velo_constructive_1973,simon_po2_1974,MR887102,rivasseau_perturbative_1991,baez_introduction_1992,jaffe_constructive_2000,MR2391806,summers_perspective_2012}}
which flourished in the 70s and 80s. 
In the two dimensional setting the existence of an
analogous object has been one of the early successes of CQFT, while in four and more dimensions (after a proper normalization) any accumulation point is necessarily Gaussian~{\cite{fernandez_random_1992}}.

The existence of an Euclidean invariant and reflection positive limit $\nu$ (plus some technical conditions) 
implies the existence of a
relativistic quantum field theory in the Minkowski space-time $\mathbb{R}^{1
+ 2}$ which satisfies the Wightman axioms~{\cite{MR0436800}}. This is a minimal set of
axioms capturing the essence of the combination of quantum mechanics and
special relativity. The translation from the commutative probabilistic setting
(Euclidean QFT) to the non-commutative Minkowski QFT setting is operated by a
set of axioms introduced by
Osterwalder--Schrader (OS)~{\cite{osterwalder_axioms_1973,osterwalder_axioms_1975}}
for the correlation functions of the measure $\nu$. These  are called Schwinger
functions or Euclidean correlation functions and shall satisfy: a
regularity axiom, a Euclidean invariance axiom, a reflection
positivity axiom, a symmetry axiom and a cluster property. \

Euclidean invariance and reflection positivity conspire against each other. 
Models which easily satisfy one property hardly satisfy the other  if they
are not Gaussian, or simple transformations thereof, see e.g.~{\cite{albeverio_HC1_2002,albeverio_hida_2009}}. Reflection
positivity itself is a property whose crucial importance for probability theory and
mathematical physics~{\cite{kotecky_reflection_2009,jaffe_reflection_2018}}
and representation theory~{\cite{neeb_reflection_2018,jorgensen_reflection_2018}} has been one of
the byproducts of the constructive effort.

\medskip

The original proof of the OS axioms, along with additional properties of the
limiting measures which are called $\Phi^4_3$
measures, is scattered in a series of works covering almost a decade.
Glimm~{\cite{glimm_boson_1968}} first proved the existence of the Hamiltonian
(with an infrared regularization) in the Minkowski setting. Then Glimm and
Jaffe~{\cite{glimm_positivity_1973}} introduced the \tmtextit{phase cell
expansion} of the regularized Schwinger functions, which revealed itself a
powerful and robust tool (albeit complex to digest) in order to handle the local
 singularities of Euclidean quantum fields and  to prove the ultraviolet stability
in finite volume (i.e. the limit $\varepsilon\to 0$ with $M$ fixed). The proof of existence of the infinite volume limit ($M \to \infty$) and the
verification of Osterwalder--Schrader
axioms
was then
completed, for $\lambda$
small and using cluster expansion methods, independently by Feldman and Osterwalder~{\cite{feldman_wightman_1976}} and by Magnen and
S\'en\'eor~{\cite{magnen_infinite_1976}}. Finally
the work of Seiler and Simon~{\cite{seiler_nelsons_1976}} allowed to extend
the existence result to all $\lambda > 0$ (this is claimed in~{\cite{MR887102}} even though we could not find a clear statement in Seiler and Simon's paper). Equations of
motion for the quantum fields were established by Feldman and
R{\c{a}}czka~{\cite{feldman77}}.

\medskip

Since this first, complete, construction, there have been several other
attempts to simplify (both technically and conceptually) the arguments and
the $\Phi^4_3$ measure has been since considered a test bed for various CQFT
techniques. There exists at least six methods of  proof: the original \emph{phase
cell method} of Glimm and Jaffe extended by Feldman and
Osterwalder~{\cite{feldman_wightman_1976}}, Magnen and
S\'en\'eor~{\cite{magnen_infinite_1976}}  and Park~{\cite{park_convergence_1977}}
(among others), the probabilistic approach of Benfatto, Cassandro, Gallavotti,
Nicol{\'o}, Olivieri, Presutti and
Schiacciatelli~{\cite{benfatto_probabilistic_1978}}, the \emph{block average method}
of Ba{\l}aban~{\cite{MR733476}} revisited by Dimock
in~{\cite{dimock_renormalization_2013_1,dimock_renormalization_2013_2,dimock_renormalization_2014_3}},
the wavelet method of Battle--Federbush~{\cite{battle_wavelets_1999}}, the
\tmtextit{skeleton inequalities method} of Brydges, Fr{\"o}hlich,
Sokal~{\cite{MR723546}}, the work of Watanabe on rotation invariance~\cite{watanabe_block_1989} via the renormalization group method of Gaw\k{e}dzki and Kupiainen~\cite{gawpolhk_edzki_asymptotic_1986}, and more recently the {renormalization group method}
 of Brydges, Dimock and Hurd~{\cite{brydges_short_1995}}.

It should be said that, apart from the Glimm--Jaffe--Feldman--Osterwalder--Magnen--S\'en\'eor
result, none of the additional constructions seems to be as complete and to
verify explicitly all the OS axioms. As Jaffe~{\cite{MR2391806}} remarks:

\begin{quotation}
  ``Not only should one give a transparent proof of the dimension $d = 3$ 
  construction, but as explained to me by Gelfand [private communication], one
  should make it sufficiently attractive that probabilists will take
  cognizance of the existence of a wonderful mathematical object.''
\end{quotation}

\medskip
The proof of Theorem~\ref{th:main} uses tools from the PDE theory as
well as recent advances in the field of \tmtextit{singular SPDEs}, without using any input 
 from traditional CQFT.
It applies to all  values of the coupling parameter $\lambda>0$ as well as to natural
extensions to  $N$-dimensional vectorial and long-range variants of the model. 
 
 Our methods are  very different from all the known constructions we
enumerated above. In particular, we do not rely on any of the standard tools
like cluster expansion or correlation inequalities or skeleton inequalities, and therefore our approach
brings a new perspective to this extensively investigated classical problem,
with respect to the removal of both ultraviolet and infrared regularizations.

Showing  invariance under translation, reflection positivity, the regularity
axiom of Osterwalder and Schrader  and the non-Gaussianity of the measure, we go a long
way (albeit not fully reaching the goal) to a complete independent construction of the $\Phi^4_3$ quantum field theory.
Furthermore, the integration by parts formula that we are able to establish leads to the hierarchy of the
Dyson--Schwinger equations for the Schwinger functions of the measure.

The key idea is to use a dynamical description of the approximate measure which
relies on an additional random source term which is Gaussian, in the spirit of the
\tmtextit{stochastic quantization} approach introduced by
Nelson~{\cite{nelson1966,MR0214150}} and Parisi and
Wu~{\cite{parisi_perturbation_1981}} (with a precursor in a technical report
of Symanzik~{\cite{Symanzik1964}}).

The concept of \emph{stochastic quantization} refers to the introduction of a reversible
stochastic dynamics which has the target measure as the invariant measure, here in
particular the $\Phi^4_d$ measure in $d$ dimensions. The rigorous study of
the stochastic quantization for the two dimensional version of the $\Phi^4$ theory
has been first initiated by Jona-Lasinio and
Mitter~{\cite{jona-lasinio_stochastic_1985}} in finite volume and by Borkar,
Chari and Mitter~{\cite{borkar_stochastic_1988}} in infinite volume. A natural
$d = 2$ local dynamics has been subsequently constructed by Albeverio and
R{\"o}ckner~{\cite{albeverio_stochastic_1991}} using Dirichlet forms in
infinite dimensions. Later on, Da Prato and
Debussche~{\cite{da_prato_strong_2003}} have shown for the first time the
existence of strong solutions to the stochastic dynamics in finite volume.
Da Prato and Debussche have introduced an innovative use of a mixture of probabilistic and
PDE techniques and constitute a landmark in the development of PDE techniques
to study stochastic analysis problems.  Similar methods have been used by McKean~\cite{mckean_1995,mckean_1995_err} and Bourgain~\cite{bourgain_invariant_1996} in the context of random data deterministic PDEs. Mourrat and
Weber~{\cite{MW17}} have subsequently shown the existence and
uniqueness of the stochastic dynamics globally in space and time. For the $d =
1$ dimensional variant, which is substantially simpler and does not require
renormalization, global existence and uniqueness have been established by
Iwata~{\cite{iwata_infinite_1987}}.

In the three dimensional setting the progress has been significantly slower due to the more
severe nature of the singularities of  solutions to the stochastic quantization
equation. Only very recently, there has been substantial progress due to the
invention of \tmtextit{regularity structures theory} by
Hairer~{\cite{hairer_theory_2014}} and \tmtextit{paracontrolled distributions}
by Gubinelli, Imkeller, Perkowski~{\cite{GIP}}. These  theories greatly
extend the pathwise approach of Da Prato and Debussche via insights coming
from Lyons' \tmtextit{rough path
theory}~{\cite{lyons_differential_1998,lyons_system_2002,lyons_differential_2007}}
and in particular the concept of \tmtextit{controlled
paths}~{\cite{gubinelli_controlling_2004,friz_course_2014}}. With these new ideas
it became possible to solve certain analytically ill-posed stochastic PDEs, including the
stochastic quantization equation for the $\Phi_3^4$ measure and the
Kardar--Parisi--Zhang equation. The first results were limited to finite
volume: local-in-time well-posedness has been established by
Hairer~{\cite{hairer_theory_2014}} and Catellier, Chouk~{\cite{CC}}.
Kupiainen~{\cite{kupiainen_renormalization_2016}} introduced a method based on the
renormalization group ideas of~\cite{gawpolhk_edzki_asymptotic_1986}. Long-time behavior has been studied by
Mourrat, Weber~{\cite{MWcomedown}}, Hairer,
Mattingly~{\cite{Hairer:2018:10.1214/17-AIHP840}} and a lattice approximation
in finite volume has been given by Hairer and Matetski~\cite{hairer_discretisations_2018} and by Zhu and Zhu~{\cite{ZZ18}}. 
Global in
space and time solutions have been first constructed by Gubinelli and
Hofmanov{\'a} in~{\cite{GH18}}. Local bounds on solutions, independent on
boundary conditions, and stretched exponential integrability have been
recently proven by Moinat and Weber~{\cite{moinat_space_time_2018}}.

However, all these advances  still fell short of giving a complete proof of the
existence of the $\Phi^4_3$ measure on the full space and of its properties.
Indeed they, including essentially all of the two dimensional
results, are principally aimed at studying the dynamics with an \tmtextit{a~priori} knowledge
of the existence and the properties of the invariant measure. For example Hairer and Matetski~\cite{hairer_discretisations_2018} use a discretization of a finite periodic domain to prove that the limiting dynamics leaves the finite volume $\Phi^4_3$ measure  invariant \emph{using} the a priori knowledge of its convergence from the paper of Brydges et al.~\cite{MR723546}. Studying the dynamics, especially globally in space and time   is still
a very complex problem which has siblings in the ever growing literature on
invariant measures for \tmtextit{deterministic} PDEs starting with the work of
Lebowitz, Rose and
Speer~{\cite{lebowitz_statistical_1988,lebowitz_statistical_1989}},
Bourgain~{\cite{bourgain_periodic_1994,bourgain_invariant_1996}}, Burq and
Tzvetkov~{\cite{burq_random_2008,burq_random_2008_1,tzvetkov_random_2016}} and with many following works (see e.g.~\cite{colliander_almost_2012, chatterjee_probabilistic_2012, nahmod_almost_2013, chatterjee_invariant_2014, benyi_probabilistic_2015}) which we cannot exhaustively review here.

The first work proposing a \tmtextit{constructive} use of the dynamics is, to our knowledge, the work
of Albeverio and Kusuoka~{\cite{albeverio_invariant_2017}}, who proved tightness of certain approximations
in a finite volume.  Inspired by this result, our aim here is to show how these
recent ideas connecting probability with PDE theory can be streamlined and
extended to recover a complete and independent proof of
existence of the $\Phi^{4}_{3}$ measure on the full space. In the same spirit see also the work of
Hairer and Iberti~{\cite{hairer_tightness_2018}} on the tightness of the 2d
Ising--Kac model.

\medskip

Soon after Hairer's seminal paper~{\cite{hairer_theory_2014}},
Jaffe~\cite{MR3392505} analyzed the stochastic quantization from the
point of view of reflection positivity and constructive QFT and concluded that
one has to necessarily take the infinite time limit to satisfy RP. Even with
global solution at hand a proof of RP from dynamics seems nontrivial and
actually the only robust tool we are aware of to prove RP is to start from
finite volume lattice Gibbs measures for which RP can be established by elementary arguments.

Taking into account these considerations, our aim is to use an equilibrium dynamics
to derive bounds which are strong enough to prove the 
tightness of the family $(\nu_{M, \varepsilon})_{M, \varepsilon}$. To be more
precise, we study a lattice approximation of the (renormalized) stochastic
quantization equation
\begin{equation}
  (\partial_t + m^2 - \Delta) \varphi + \lambda \varphi^3 - \infty \varphi =
  \xi, \qquad (t, x) \in \mathbb{R}_+ \times \mathbb{R}^3, \label{eq:P4}
\end{equation}
where $\xi$ is a space-time white noise on $\mathbb{R}^3$. The lattice
dynamics is a system of stochastic differential equation which is globally
well-posed and has $\nu_{M, \varepsilon}$ as its unique invariant measure. We
can therefore consider its stationary solution $\varphi_{M, \varepsilon}$
having at each time the law $\nu_{M, \varepsilon}$. We introduce a suitable
decomposition together with an energy method in the framework of weighted
Besov spaces. This allows us, on the one hand, to track down and renormalize
the short scale singularities present in the model as $\varepsilon \rightarrow
0$, and on the other hand, to control the growth of the solutions as $M \to
\infty$. As a result we obtain uniform bounds  which allow us to pass to the
limit in the weak topology of probability measures.

The details of the renormalized energy method rely on recent developments in
the analysis of singular PDEs. In order to make the paper accessible to a wide audience with some PDE background 
 we  implement
 renormalization using the paracontrolled calculus of~{\cite{GIP}} which is based on Bony's paradifferential operators~\cite{bony_calcul_1981, meyer_remarques_1981, BCD}.
 We  also rely on some tools from the paracontrolled analysis in weigthed Besov spaces which we developed in~{\cite{GH18}} and on the results of Martin and Perkowski~{\cite{MP17}} on Besov spaces on the lattice.

%

\begin{remark} Let us comment in detail on  specific aspects of our proof.
  \begin{enumerate}
  
   \item The method we use here  differs from the approach
of~{\cite{GH18}} in that we are initially less concerned with the continuum
dynamics itself. We do not try to obtain estimates for strong solutions and
rely instead on certain cancellations in the energy estimate that permit to
significantly simplify the proof. The resulting bounds are sufficient to
provide a rather clear picture of any limit measure as well as some of its physical
properties. In contrast, in {\cite{GH18}} we provided a
detailed control of the dynamics {\eqref{eq:P4}} (in stationary or
non-stationary situations) at the price of a more involved analysis.
Section~\ref{s:estim} of the present paper could in principle be replaced by
the corresponding analysis of {\cite{GH18}}. However the adaptation of that
analysis to the lattice setting (without which we do not know how to prove RP) would anyway require the further preparatory work which constitutes a large fraction of the present paper. Similarly, the recent results of Moinat and Weber~{\cite{moinat_space_time_2018}} (which appeared after we completed a first version of this paper) can be conceivably used to replace a part of Section~\ref{sec:tight}.

    \item The stretched exponential integrability
    in~{\eqref{eq:exp-int-intro}} is also discussed in the work of Moinat and
    Weber~{\cite{moinat_space_time_2018}} (using different norms) and it is
    sufficient to prove the original regularity axiom of Osterwalder and Schrader but not its formulation given in the book of Glimm and Jaffe~\cite{MR887102}.
    
    \item The Dyson--Schwinger equations were  first derived by Feldman
    and R{\c{a}}czka~{\cite{feldman77}} using the results of Glimm,
    Jaffe, Feldman and Osterwalder.
    
    \item As already noted by Albeverio, Liang and Zegarlinski
    {\cite{albeverio_remark_2006}} on the formal level, the integration by
    parts formula gives rise to a cubic term which cannot be interpreted as a
    random variable under the $\Phi^4_3$ measure. Therefore, the crucial
    question that remained unsolved until now is how to make sense of this
    critical term as a well-defined probabilistic object. In the present
    paper, we obtain fine estimates on the approximate stochastic quantization
    equation and construct a coupling of the stationary solution to the
    continuum $\Phi^4_3$ dynamics and the Gaussian free field. This leads to a
    detailed description of the renormalized cubic term as a genuine random
    space-time distribution. Moreover, we approximate this term in the spirit
    of the operator product expansion.
    
    \item To the best of our knowledge, our work provides the first rigorous
    proof of a general integration by parts formula with an exact formula for
    the renormalized cubic term. In addition, the method applies to arbitrary
    values of the coupling constant $\lambda \geqslant 0$ if $m^2 > 0$ and
    $\lambda > 0$ if $m^2 \leqslant 0$  and we state the precise dependence of our
    estimates on $\lambda$. In particular, we show that our energy bounds are
    uniform over $\lambda$ in every bounded subset of $[0, \infty)$ provided $m^2 >0$ (see Remark~\ref{rem:neg-mass}). Let us recall  that for some $m^{2}=m_{c}^{2}(\lambda)$ 
     the physical mass of the continuum theory is zero and it is said that the model is critical. Existence of such a critical point was shown in \cite[Section 9, Part (4)]{MR723546}. We note that this case is included in our construction, even though we are not able to locate it since we do not have control over correlations. Its large scale limit  is conjectured to correspond to the Ising conformal field theory, recently  actively studied in \cite{MR3942977} using the conformal bootstrap approach.
    
%
    
    \item By essentially the same arguments, we are able to treat the vector
    version of the model, where the scalar field $\varphi : \mathbb{R}^3
    \rightarrow \mathbb{R}$ is replaced by a vector valued one $\varphi :
    \mathbb{R}^3 \rightarrow \mathbb{R}^N$ for some $N \in \mathbb{N}$ and the
    measures $\nu_{M, \varepsilon}$ are given by a similar expression
    as~{\eqref{eq:gibbs}}, where the norm $| \varphi |$ is understood as the
    Euclidean norm in $\mathbb{R}^N$.

    \item Our proof also readily extends to the \emph{fractional} variant of $\Phi^4_3$ where
     the base Gaussian measure is obtained from the fractional Laplacian $(-\Delta)^\gamma$ with $\gamma\in(21/22,1)$ (see Section~\ref{sec:fractional} for details). In general this model is sub-critical for $\gamma\in(3/4,1)$ and in the mass-less case 
     it has recently attracted some interest since  it is \emph{bootstrappable}~\cite{poland_conformal_2019, behan_bootstrapping_2019}. 
     
  \end{enumerate}
\end{remark}

To conclude this introductory part, let us compare our result with other
constructions of the $\Phi^4_3$ field theory. The most straightforward and
simplest available proof has been given by Brydges, Fr{\"o}hlich and
Sokal~{\cite{MR723546}} using skeleton  and correlation inequalities.
All the other methods we cited above employ technically involved machineries
and various kinds of expansions (they are however able to obtain very strong
information about the model in the weakly-coupled regime, i.e. when $\lambda$ is small). Compared to the
existing methods, ours bears similarity in conceptual simplicity to that
of~{\cite{MR723546}}, with some advantages and some disadvantages. Both works
construct the continuum $\Phi^4_3$ theory as a subsequence limit of lattice
theories and the rotational invariance remains unproven. The main difference
is that~{\cite{MR723546}} relies on correlation inequalities. On the
one hand, this restricts the applicability to weak couplings and only models with
$N = (0,) 1, 2$ components (note that the $N=0$ models have a meaning only in their formalism but not in ours). But, on the other hand, it allows to establish bounds
on the decay of correlation functions, which we do not have. However, our results hold for every
value of $\lambda > 0$ and $m^2 \in \mathbb{R}$ while the results
in~{\cite{MR723546}} work only in the so-called ``single phase region'',
which  corresponds to  $m^2>m_{c}^{2}(\lambda)$.

\medskip
Our work is intended as a first step in the direction of using PDE methods in the study of
Euclidean QFTs and large scale properties of statistical mechanical models. Another related attempt is the variational approach developed in~\cite{barashkov_variational_2018} for the finite volume $\Phi^4_3$ measure.
As far as the present paper is concerned the main open problem is to establish rotational invariance and to give more information on the limiting measures, in particular to establish uniqueness for small $\lambda$. It is not clear how to deduce anything about correlations from the dynamics but it seems to be a very interesting and challenging problem.

\paragraph{Plan.} The paper is organized as follows. Section~\ref{s:not} gives a summary of
notation used throughout the paper, Section~\ref{s:strat} presents the main ideas of our strategy and Section~\ref{sec:tight}, Section~\ref{s:ax} and
Section~\ref{s:sd} are devoted to the main results. First, in Section~\ref{sec:tight}
we construct the Euclidean quantum field theory as a limit of the approximate
Gibbs measures $\nu_{M, \varepsilon}$. To this end, we introduce the lattice
dynamics together with its decomposition. The main energy estimate is
established in Theorem~\ref{th:energy-estimate} and consequently the desired
tightness as well as moment bounds are proven in Theorem~\ref{thm:main}.  In
Section~\ref{s:exp} we establish finite stretched exponential moments. Consequently, in
Section~\ref{s:ax} we verify the translation invariance and reflection positivity, the regularity axiom
and non-Gaussianity of any limit measure. Section~\ref{s:sd} is devoted to the
integration by parts formula and the Dyson--Schwinger equations. In Section~\ref{sec:fractional} we discuss the extension of 
our results to a long-range version of the $\Phi^4_3$ model.
Finally, in Appendix~\ref{s:app} we collect a number of technical results
needed in the main body of the paper.

\paragraph{Acknowledgement.} The authors would like to thank the Isaac Newton
Institute for Mathematical Sciences for support and hospitality during the
programme Scaling limits, rough paths, quantum field theory when work on this
paper was undertaken. In particular, we are grateful to Adelmalek Abdesselam, Sergio Albeverio,
David Brydges, J{\"u}rg Fr{\"o}hlich, Stefan Hollands, Seiichiro Kusuoka and
Pronob Mitter for stimulating discussions. We are also deeply grateful to the anonymous referees for their impressively detailed comments on the CQFT literature and on the relations between various set of axioms for Euclidean correlation functions which helped us to  substantially improve this paper and our knowledge.

This work was supported by EPSRC
Grant Number EP/R014604/1. M. G. is partially supported by the German Research
Foundation (DFG) via CRC 1060.

\section{Notation}

\label{s:not}Within this paper we are concerned with the $\Phi^4_3$ model in
discrete as well as continuous setting. In particular, we denote by
$\Lambda_{\varepsilon} = (\varepsilon \mathbb{Z})^d$ for $\varepsilon = 2^{-
N}$, $N \in \mathbb{N}_0$, the rescaled lattice $\mathbb{Z}^d$ and by
$\Lambda_{M, \varepsilon} = \varepsilon \mathbb{Z}^d \cap \mathbb{T}^d_M =
\varepsilon \mathbb{Z}^d \cap \left[ - \frac{M}{2}, \frac{M}{2} \right)^d$ its
periodic counterpart of size $M > 0$ such that  $M/(2\varepsilon)\in\mathbb{N}$. For notational simplicity, we use the
convention that the case $\varepsilon = 0$ always refers to the continuous
setting. For instance, we denote by $\Lambda_0$ the full space $\Lambda_0
=\mathbb{R}^d$ and by $\Lambda_{M, 0}$ the continuous torus $\Lambda_{M, 0}
=\mathbb{T}^d_M$. With the slight abuse of notation, the parameter
$\varepsilon$ is always taken either of the form $\varepsilon = 2^{- N}$ for
some $N \in \mathbb{N}_0$, $N \geqslant N_0$, for certain $N_0 \in
\mathbb{N}_0$ that will be chosen as a consequence of Lemma~\ref{lem:equiv}
below, or $\varepsilon = 0$. Various proofs below will be formulated generally
for $\varepsilon \in \mathcal{A} \assign \{ 0, 2^{- N} ; N \in \mathbb{N}_0, N
\geqslant N_0 \}$ and it is understood that the case $\varepsilon = 0$ or
alternatively $N = \infty$ refers to the continuous setting. All the proportionality constants, unless explicitly signalled, will be independent of $M,\varepsilon,\lambda,m^2$. We will track the explicit dependence on $\lambda$ as far as possible and signal when the constant depends on the value of $m^2>0$. 

\medskip

For $f \in \ell^1
(\Lambda_{\varepsilon})$ and $g \in L^1 (\hat{\Lambda}_{\varepsilon})$, respectively, we define the Fourier and the inverse Fourier transform as
\[ \mathcal{F} f (k) = \varepsilon^d \sum_{x \in \Lambda_{\varepsilon}} f (x)
   e^{- 2 \pi i k \cdummy x}, \qquad \mathcal{F}^{- 1} g (x) = \int_{(\varepsilon^{- 1} \mathbb{T})^d} g (k)
   e^{2 \pi i k \cdummy x} \mathd k, \]
   where $k \in (\varepsilon^{- 1} \mathbb{T})^d
   \backassign \hat{\Lambda}_{\varepsilon}$ and $x \in \Lambda_{\varepsilon}$.
These definitions can be extended to discrete Schwartz distributions in a
natural way, we refer to {\cite{MP17}} for more details. In general, we do not
specify on which lattice the Fourier transform is taken as it will be clear
from the context.

Consider a smooth dyadic partition of unity $(\varphi_j)_{j \geqslant - 1}$
such that $\varphi_{- 1}$ is supported in a ball around $0$ of radius
$\frac{1}{2}$, $\varphi_0$ is supported in an annulus, $\varphi_j (\cdummy) =
\varphi_0 (2^{- j} \cdummy)$ for $j \geqslant 0$ and if $| i - j | > 1$ then
$\tmop{supp} \varphi_i \cap \tmop{supp} \varphi_j = \emptyset$. For the
definition of Besov spaces on the lattice $\Lambda_{\varepsilon}$ for
$\varepsilon = 2^{- N}$, we introduce a suitable periodic partition of unity
on $\hat{\Lambda}_{\varepsilon}$ as follows
\begin{equation}
  \varphi^{\varepsilon}_j (k) \assign \left\{ \begin{array}{lll}
    \varphi_j (k), &  & j < N - J,\\
    1 - \sum_{j < N - J} \varphi_j (k), &  & j = N - J,
  \end{array} \right. \label{eq:p1}
\end{equation}
where $k \in \hat{\Lambda}_{\varepsilon}$ and the parameter $J \in
\mathbb{N}_0$, whose precise value will be chosen below independently on
$\varepsilon \in \mathcal{A}$, satisfies $0 \leqslant N - J \leqslant
J_{\varepsilon} \assign \inf \{ j : \tmop{supp} \varphi_j \not\subseteq
[-\varepsilon^{- 1}/2,\varepsilon^{- 1}/2 )^d \} \rightarrow \infty$ as
$\varepsilon \rightarrow 0$. We note that by construction there exists $\ell
\in \mathbb{Z}$ independent of $\varepsilon = 2^{- N}$ such that
$J_{\varepsilon} = N - \ell$.

Then {\eqref{eq:p1}} yields a periodic partition of unity on
$\hat{\Lambda}_{\varepsilon}$. The reason for choosing the upper index as $N -
J$ and not the maximal choice $J_{\varepsilon}$ will become clear in Lemma
\ref{lem:equiv} below, where it allows us to define suitable localization
operators needed for our analysis. The choices of parameters $N_0$ and $J$ are
related in the following way: A given partition of unity $(\varphi_j)_{j
\geqslant - 1}$ determines the parameters $J_{\varepsilon}$ in the form
$J_{\varepsilon} = N - \ell$ for some $\ell \in \mathbb{Z}$. By the condition
$N - J \leqslant J_{\varepsilon}$ we obtain the first lower bound on $J$. Then
Lemma~\ref{lem:equiv} yields a (possibly larger) value of $J$ which is fixed
throughout the paper. Finally, the condition $0 \leqslant N - J$ implies
the necessary lower bound $N_0$ for $N$, or alternatively the upper bound for
$\varepsilon = 2^{- N} \leqslant 2^{- N_0}$ and defines the set $\mathcal{A}$.
We stress that once the parameters $J, N_0$ are chosen, they remain fixed
throughout the paper.

Remark that according to our convention, $(\varphi^0_j)_{j \geqslant - 1}$
denotes the original partition of unity $(\varphi_j)_{j \geqslant - 1}$ on
$\mathbb{R}^d$, which can be also read from {\eqref{eq:p1}} using the
fact that for $\varepsilon = 0$ we have $J_{\varepsilon} = \infty$.

Now we may define the Littlewood--Paley blocks for distributions on
$\Lambda_{\varepsilon}$ by
\[ \Delta_j^{\varepsilon} f \assign \mathcal{F}^{- 1} (\varphi_j^{\varepsilon}
   \mathcal{F} f), \]
which leads us to the definition of weighted Besov spaces. Throughout the paper, $\rho$ denotes a polynomial weight 
of the form  \begin{equation}\label{eq:weight}
\rho (x) =
\langle h x \rangle^{- \nu} = (1 + |h  x |^2)^{- \nu / 2}
\end{equation}
 for some $\nu \geqslant
0$ and $h>0$. The constant $h$ will be fixed below in Lemma \ref{lemma:bounds-rhs1} in order to produce a small bound for certain terms.
 Such weights satisfy the admissibility condition $\rho(x)/\rho(y)\lesssim \rho^{-1}(x-y)$
for all $ x, y
   \in \mathbb{R}^d . $
For $\alpha \in \mathbb{R}$, $p, q \in [1, \infty]$ and $\varepsilon \in [0,
1]$ we define the weighted Besov spaces on $\Lambda_{\varepsilon}$ by the norm
\[ \| f \|_{B^{\alpha, \varepsilon}_{p, q} (\rho)} = \Bigg( \sum_{- 1
   \leqslant j \leqslant N - J} 2^{\alpha j q} \| \Delta_j^{\varepsilon} f
   \|_{L^{p, \varepsilon} (\rho)}^q \Bigg)^{1 / q} = \Bigg( \sum_{- 1
   \leqslant j \leqslant N - J} 2^{\alpha j q} \| \rho \Delta_j^{\varepsilon}
   f \|_{L^{p, \varepsilon}}^q \Bigg)^{1 / q}, \]
where $L^{p, \varepsilon}$ for $\varepsilon \in \mathcal{A} \setminus \{ 0 \}$
stands for the $L^p$ space on $\Lambda_{\varepsilon}$ given by the norm
\[ \| f \|_{L^{p, \varepsilon}} = \Bigg( \varepsilon^d \sum_{x \in
   \Lambda_{\varepsilon}} | f (x) |^p \Bigg)^{1 / p} \]
(with the usual modification if $p = \infty$). Analogously, we may define the weighted Besov spaces for explosive polynomial weights of the form $\rho^{-1}$. Note that if $\varepsilon = 0$
then $B^{\alpha, \varepsilon}_{p, q} (\rho)$ is the classical weighted Besov
space $B^{\alpha}_{p, q} (\rho)$. In the sequel, we also employ the following
notations
\[ \CC^{\alpha, \varepsilon} (\rho) \assign B^{\alpha, \varepsilon}_{\infty,
   \infty} (\rho), \qquad H^{\alpha, \varepsilon} (\rho) \assign B^{\alpha,
   \varepsilon}_{2, 2} (\rho) . \]
In Lemma~\ref{lem:equiv2}  we show that one can pull the weight
inside the Littlewood--Paley blocks in the definition of the weighted Besov
spaces. Namely, under suitable assumptions on the weight that are satisfied by
polynomial weights we have
$ \| f \|_{B^{\alpha, \varepsilon}_{p, q} (\rho)} \sim \| \rho f
   \|_{B^{\alpha, \varepsilon}_{p, q}} $
in the sense of equivalence of norms, uniformly in $\varepsilon$.   We define the duality product on
$\Lambda_{\varepsilon}$ by
\[ \langle f, g \rangle_{\varepsilon} \assign \varepsilon^d \sum_{x \in
   \Lambda_{\varepsilon}} f (x) g (x)  \]
and Lemma~\ref{lem:dual2} shows that $B^{- \alpha, \varepsilon}_{p', q'} (\rho^{- 1})$ is
included in the topological dual of $B^{\alpha, \varepsilon}_{p, q} (\rho)$
for conjugate exponents $p, p'$ and $q, q'$.

\medskip

We employ the tools from paracontrolled calculus as
introduced in {\cite{GIP}}, the reader is also referred to {\cite{BCD}} for
further details. We shall  freely use the decomposition $f g = f \prec g +
f \circ g + f \succ g$, where $f \succ g = g \succ f$ and $f \circ g$,
respectively, stands for the paraproduct of $f$ and $g$ and the corresponding
resonant term, defined in terms of Littlewood--Paley decomposition. More
precisely, for $f, g \in \mathcal{S}' (\Lambda_{\varepsilon})$ we let
\[ f \prec g \assign \sum_{1 \leqslant i, j \leqslant N - J, i < j - 1}
   \Delta^{\varepsilon}_i f \Delta^{\varepsilon}_j g, \qquad f \circ g \assign
   \sum_{1 \leqslant i, j \leqslant N - J, i \sim j} \Delta^{\varepsilon}_i f
   \Delta^{\varepsilon}_j g. \]
   We also employ the notations $f\preccurlyeq g:= f\prec g+f\circ g$ and $f\Join g:=f\prec g+f\succ g$.
For notational simplicity, we do not stress the dependence of the paraproduct
and the resonant term on $\varepsilon$ in the sequel. These paraproducts
satisfy the usual estimates uniformly in $\varepsilon$, see e.g.
{\cite{MP17}}, Lemma~4.2, which can be naturally extended to general
$B^{\alpha, \varepsilon}_{p, q} (\rho)$ Besov spaces as in {\cite{MW17}},
Theorem~3.17.

\medskip

Throughout the paper we assume that $m^{2}>0$ and we only discuss in Remark \ref{rem:neg-mass} how to treat the case of $m^{2}\leqslant0$. In addition,  we are only concerned with the 3 dimensional setting and let $d = 3$. We denote by $\Delta_{\varepsilon}$ the discrete Laplacian on $\Lambda_{\varepsilon}$
given by
\[ \Delta_{\varepsilon} f (x) = \varepsilon^{- 2} \sum_{i = 1}^d (f (x +
   \varepsilon e_i) - 2 f (x) + f (x - \varepsilon e_i)), \qquad x \in
   \Lambda_{\varepsilon}, \]
where $(e_i)_{i = 1, \ldots, d}$ is the canonical basis of $\mathbb{R}^d$.  It
can be checked by a direct computation that the integration by parts formula
\[ \langle \Delta_{\varepsilon} f, g \rangle_{M, \varepsilon} = - \langle
   \nabla_{\varepsilon} f, \nabla_{\varepsilon} g \rangle_{M, \varepsilon} = -
   \varepsilon^d \sum_{x \in \Lambda_{M, \varepsilon}} \sum_{i = 1}^d \frac{f
   (x + \varepsilon e_i) - f (x)}{\varepsilon}  \frac{g (x + \varepsilon e_i)
   - g (x)}{\varepsilon} \]
holds for the discrete gradient
\[ \nabla_{\varepsilon} f (x) = \left( \frac{f (x + \varepsilon e_i) - f
   (x)}{\varepsilon} \right)_{i = 1, \ldots, d} . \]
We let $\Q_{\varepsilon} \assign m^{2} - \Delta_{\varepsilon}$,  $\LL_{\varepsilon} \assign \partial_t +
\Q_{\varepsilon}$ and we write $\LL$ for the continuum analogue of $\LL_{\varepsilon}$.  We let $\LL_{\varepsilon}^{- 1}$ to be the inverse of
$\LL_{\varepsilon}$ on $\Lambda_{\varepsilon}$ such that $\LL_{\varepsilon}^{-
1} f = v$ is a solution to $\LL_{\varepsilon} v = f$, $v (0) = 0.$

\section{Overview of the strategy}
\label{s:strat}

With the goals and notations being set, let us now outline the main steps of our strategy.

\paragraph{Lattice dynamics.}

For fixed parameters $\varepsilon \in \mathcal{A}, M > 0$, we consider a stationary
solution $\varphi_{M, \varepsilon}$ to the discrete stochastic quantization
equation
\begin{equation}
  \mathscr{L} \hspace{.2em}_{\varepsilon} \varphi_{M, \varepsilon} + \lambda
  \varphi_{M, \varepsilon}^3 + (- 3 \lambda a_{M, \varepsilon} + 3
  \lambda^2 b_{M, \varepsilon}) \varphi_{M, \varepsilon} = \xi_{M,
  \varepsilon}, \qquad x \in \Lambda_{M, \varepsilon}, \label{eq:moll}
\end{equation}
whose law at every time $t \geqslant 0$ is given by the Gibbs measure
{\eqref{eq:gibbs}}.
Here $\xi_{M, \varepsilon}$ is a discrete approximation of
a space-time white noise $\xi$ on $\mathbb{R}^{d}$ constructed as follows: Let $\xi_M$ denote its periodization on $\mathbb{T}^d_M$
given by
\[ \xi_M (h) \assign \xi (h_M), \qquad \tmop{where} \quad h_M (t, x) \assign
   \tmmathbf{1}_{\left[ - \frac{M}{2}, \frac{M}{2} \right)^d} (x) \sum_{y \in
   M\mathbb{Z}^d} h (t, x + y), \]
where $h\in L^{2}(\mathbb R\times\mathbb R^{d})$ is a test function, and define the corresponding spatial discretization by
\[ \xi_{M, \varepsilon} (t, x) \assign \varepsilon^{- d} \langle \xi_M (t,
   \cdummy), \tmmathbf{1}_{| \cdummy - x | \leqslant \varepsilon / 2} \rangle,
   \qquad (t, x) \in \mathbb{R} \times \Lambda_{M, \varepsilon} . \]
Then {\eqref{eq:moll}} is a finite-dimensional SDE in a gradient form and it has a (unique) invariant measure $\nu_{M, \varepsilon}$ given by~\eqref{eq:gibbs}. Indeed, the global existence of solutions can be proved along the lines of Khasminskii nonexplosion test \cite[Theorem 3.5]{khasminskii2011stochastic} whereas invariance of the measure~\eqref{eq:gibbs} follows from \cite[Theorem 2]{Zab89}.

Recall that due to the irregularity of the
space-time white noise in dimension $3$, a solution to the limit problem
{\eqref{eq:P4}} can only exist as a distribution. Consequently, since products
of distributions are generally not well-defined it is necessary to make sense
of the cubic term. This forces us to introduce a mass renormalization via
constants $a_{M, \varepsilon}, b_{M, \varepsilon} \geqslant 0$ in
{\eqref{eq:moll}} which shall be suitably chosen in order to compensate the
ultraviolet divergencies. In other words, the additional linear term shall
introduce the correct counterterms needed to renormalize the cubic power and
to derive estimates uniform in both parameters $M, \varepsilon$. To this end,
$a_{M, \varepsilon}$ shall diverge linearly whereas $b_{M, \varepsilon}$
logarithmically and these are of course the same divergencies as those
appearing in the other approaches, see e.g. Chapter 23 in {\cite{MR887102}}.

\paragraph{Energy method in a nutshell.}Our aim is to apply the so-called
energy method, which is one of the very basic approaches in the PDE theory. It
relies on testing the equation by the solution itself and estimating all the
terms. To explain the main idea, consider a toy model
\[ \LL u + \lambda u^3 = f, \qquad x \in \mathbb{R}^3,
\]
driven by a sufficiently regular forcing $f$ such that  the solution is
smooth and there are no difficulties in defining the cube. Testing the
equation by $u$ and integrating the Laplace term by parts leads to
\[ \frac{1}{2} \partial_t \| u \|_{L^2}^2 + m^2 \| u \|_{L^2}^2 + \| \nabla u
   \|_{L^2}^2 + \lambda \| u \|_{L^4}^4 = \langle f, u \rangle . \]
Now, there are several possibilities to estimate the right hand side using
duality and Young's inequality, namely,
\[ \langle f, u \rangle \leqslant \left\{ \begin{array}{l}
     \| f \|_{L^2} \| u \|_{L^2} \leqslant C_{ m^2} \| f \|_{L^2}^2 +
     \frac12 m^2 \| u \|_{L^2}^2\\
     \| f \|_{L^{4 / 3}} \| u \|_{L^4} \leqslant C \lambda^{- 1 / 3}
     \| f \|_{L^{4 / 3}}^{4 / 3} + \frac12 \lambda \| u \|_{L^4}^4\\
     \| f \|_{H^{- 1}} \| u \|_{H^1} \leqslant C_{m^2} \| f \|_{H^{-
     1}}^2 + \frac12 (m^2 \| u \|_{L^2}^2 + \| \nabla u \|_{L^2}^2)
   \end{array} . \right. \]
This way, the dependence on $u$ on the right hand side can be absorbed into
the good terms on the left hand side. If in
addition $u$ was stationary hence in particular $t \mapsto \mathbb{E} \| u (t)
\|_{L^2}^2$ is constant, then we obtain
\[ m^2 \mathbb{E} \| u (t) \|_{L^2}^2 +\mathbb{E} \| \nabla u (t) \|_{L^2}^2 +
   \lambda \mathbb{E} \| u (t) \|_{L^4}^4 \leqslant \left\{ \begin{array}{l}
     C_{m^2} \| f \|_{L^2}^2\\
     C \lambda^{- 1 / 3} \| f \|_{L^{4 / 3}}^{4 / 3}\\
     C_{m^2} \| f \|_{H^{- 1}}^2
   \end{array} . \right. \]

To summarize, using the dynamics we are able to obtain moment bounds for the
invariant measure that depend only on the forcing $f$. Moreover, we also see
the behavior of the estimates with respect to the coupling constant $\lambda$.
Nevertheless, even though using the $L^4$-norm of $u$ introduces a blow up for
$\lambda \rightarrow 0$, the right hand side $f$ in our energy estimate below
will always contain certain power of $\lambda$ in order to cancel this blow up
and to obtain bounds that are uniform as $\lambda \rightarrow 0$.

\paragraph{Decomposition and estimates.} Since the forcing $\xi$ on the right
hand side of {\eqref{eq:P4}} does not possess sufficient regularity, the
energy method cannot be applied directly. Following the usual approach within
the field of singular SPDEs, we shall find a suitable decomposition of the
solution $\varphi_{M, \varepsilon}$, isolating parts of different regularity.
In particular, since the equation is subcritical in the sense of Hairer
{\cite{hairer_theory_2014}} (or superrenormalizable in the language of quantum
field theory), we expect the nonlinear equation {\eqref{eq:P4}} to be a
perturbation of the linear problem $ \LL X = \xi .$
This singles out the most irregular part of the limit field $\varphi$. Hence on
the approximate level we set $\varphi_{M, \varepsilon} = X_{M, \varepsilon} +
\eta_{M, \varepsilon}$ where $X_{M, \varepsilon}$ is a stationary solution to
\begin{equation}
 \LL_{\varepsilon} X_{M,\varepsilon} = \xi_{M,\varepsilon} , \label{eq:X}
\end{equation}
and the remainder
$\eta_{M, \varepsilon}$ is expected to be more regular.

To see if it is indeed the case we plug our decomposition into
{\eqref{eq:moll}} to obtain
\begin{equation}
 \LL_{\varepsilon} \eta_{M, \varepsilon} + 3
  \lambda^2 b_{M, \varepsilon} \varphi_{M, \varepsilon} + \lambda \llbracket
  X_{M, \varepsilon}^3 \rrbracket + \lambda 3 \eta_{M, \varepsilon} \llbracket
  X_{M, \varepsilon}^2 \rrbracket + \lambda 3 \eta_{M, \varepsilon}^2 X_{M,
  \varepsilon} + \lambda \eta_{M, \varepsilon}^3 = 0. \label{eq:eta}
\end{equation}
Here $\llbracket X^2_{M, \varepsilon} \rrbracket$ and $\llbracket X^3_{M,
\varepsilon} \rrbracket$ denote the second and third Wick power of the
Gaussian random variable $X_{M, \varepsilon}$ defined by
\begin{equation}\label{eq:X2X3}
\llbracket X^2_{M, \varepsilon} \rrbracket \assign X^2_{M, \varepsilon} -
   a_{M, \varepsilon}, \qquad \llbracket X^3_{M, \varepsilon} \rrbracket
   \assign X^3_{M, \varepsilon} - 3 a_{M, \varepsilon} X_{M, \varepsilon},
   \end{equation}
where $a_{M, \varepsilon} \assign \mathbb{E} [X^2_{M, \varepsilon} (t)]$ is
independent of $t$ due to stationarity. It can be shown by direct computations
that appeared already in a number of works (see {\cite{CC}},
{\cite{hairer_theory_2014}}, {\cite{hairer_regularity_2015}},
{\cite{mourrat_construction_2016}}) that $\llbracket X^2_{M, \varepsilon}
\rrbracket$ is bounded
uniformly in $M, \varepsilon$ as a continuous stochastic process with values
in the weighted Besov space $\mathscr{C} \hspace{.1em}^{- 1 - \kappa,\varepsilon}
(\rho^{\sigma})$ for every $\kappa, \sigma > 0$, whereas  $\llbracket X^3_{M, \varepsilon} \rrbracket$ can only be constructed as a space-time distribution. In addition,
they converge to the Wick power $\llbracket X^2 \rrbracket$ and $\llbracket
X^3 \rrbracket$ of $X$. In other words, the
linearly growing renormalization constant $a_{M, \varepsilon}$ gives
counterterms needed for the Wick ordering.

Note that $X$ is a continuous
stochastic process with values in $\mathscr{C} \hspace{.1em}^{- 1 / 2 -
\kappa} (\rho^{\sigma})$ for every $\kappa, \sigma > 0$. This limits the
regularity that can be obtained for the approximations $X_{M, \varepsilon}$
uniformly in $M, \varepsilon$. Hence the most irregular term in
{\eqref{eq:eta}} is the third Wick power and by Schauder estimates we expect
$\eta_{M, \varepsilon}$ to be 2 degrees of regularity better. Namely, we
expect uniform bounds for $\eta_{M, \varepsilon}$ in $\mathscr{C}
\hspace{.1em}^{1 / 2 - \kappa} (\rho^{\sigma})$ which indeed verifies our
presumption that $\eta_{M, \varepsilon}$ is more regular than $\varphi_{M,
\varepsilon}$. However, the above decomposition introduced new products in
{\eqref{eq:eta}} that are not well-defined under the above discussed uniform
bounds. In particular, both $\eta_{M, \varepsilon} \llbracket X_{M,
\varepsilon}^2 \rrbracket$ and $\eta_{M, \varepsilon}^2 X_{M, \varepsilon}$ do
not meet the condition that the sum of their regularities is strictly
positive, which is a convenient sufficient  condition  for a product of two distributions
to be analytically well-defined.

In order  to continue the decomposition in the same spirit and to cancel the most irregular term in {\eqref{eq:eta}}, namely,
$\llbracket X^3_{M, \varepsilon} \rrbracket$. The usual way, which can be found
basically in all the available works on the stochastic quantization (see e.g.
in \ {\cite{CC}}, {\cite{GH18}}, {\cite{hairer_theory_2014}},
{\cite{hairer_regularity_2015}}, {\cite{MWcomedown}}) is therefore to define
$X_{M, \varepsilon}^{\!\resizebox{0.6em}{!}{
\begin{tikzpicture}
\pgfpathmoveto{\pgfqpoint{0cm}{-0.035cm}}
\pgfpathlineto{\pgfqpoint{1.376cm}{-0.035cm}}
\pgfpathlineto{\pgfqpoint{1.376cm}{1.552cm}}
\pgfpathlineto{\pgfqpoint{0cm}{1.552cm}}
\pgfpathclose
\pgfusepath{clip}
\begin{pgfscope}
\begin{pgfscope}
\pgfpathmoveto{\pgfqpoint{0cm}{-0.035cm}}
\pgfpathlineto{\pgfqpoint{1.376cm}{-0.035cm}}
\pgfpathlineto{\pgfqpoint{1.376cm}{1.552cm}}
\pgfpathlineto{\pgfqpoint{0cm}{1.552cm}}
\pgfpathclose
\pgfusepath{clip}
\begin{pgfscope}
\begin{pgfscope}
\pgfsetdash{}{0cm}
\pgfsetlinewidth{0.818mm}
\pgfsetroundcap
\pgfsetroundjoin
\pgfsetmiterlimit{7.0}
\definecolor{eps2pgf_color}{gray}{0}\pgfsetstrokecolor{eps2pgf_color}\pgfsetfillcolor{eps2pgf_color}
\pgfpathmoveto{\pgfqpoint{0.117cm}{1.421cm}}
\pgfpathlineto{\pgfqpoint{0.682cm}{0.671cm}}
\pgfpathlineto{\pgfqpoint{1.246cm}{1.421cm}}
\pgfusepath{stroke}
\end{pgfscope}
\definecolor{eps2pgf_color}{gray}{0}\pgfsetstrokecolor{eps2pgf_color}\pgfsetfillcolor{eps2pgf_color}
\pgfpathmoveto{\pgfqpoint{0.273cm}{1.395cm}}
\pgfpathcurveto{\pgfqpoint{0.273cm}{1.432cm}}{\pgfqpoint{0.259cm}{1.467cm}}{\pgfqpoint{0.233cm}{1.492cm}}
\pgfpathcurveto{\pgfqpoint{0.207cm}{1.518cm}}{\pgfqpoint{0.173cm}{1.532cm}}{\pgfqpoint{0.137cm}{1.532cm}}
\pgfpathcurveto{\pgfqpoint{0.1cm}{1.532cm}}{\pgfqpoint{0.066cm}{1.518cm}}{\pgfqpoint{0.04cm}{1.492cm}}
\pgfpathcurveto{\pgfqpoint{0.014cm}{1.467cm}}{\pgfqpoint{0cm}{1.432cm}}{\pgfqpoint{0cm}{1.395cm}}
\pgfpathcurveto{\pgfqpoint{0cm}{1.359cm}}{\pgfqpoint{0.014cm}{1.324cm}}{\pgfqpoint{0.04cm}{1.299cm}}
\pgfpathcurveto{\pgfqpoint{0.066cm}{1.273cm}}{\pgfqpoint{0.1cm}{1.258cm}}{\pgfqpoint{0.137cm}{1.258cm}}
\pgfpathcurveto{\pgfqpoint{0.173cm}{1.258cm}}{\pgfqpoint{0.207cm}{1.273cm}}{\pgfqpoint{0.233cm}{1.299cm}}
\pgfpathcurveto{\pgfqpoint{0.259cm}{1.324cm}}{\pgfqpoint{0.273cm}{1.359cm}}{\pgfqpoint{0.273cm}{1.395cm}}
\pgfusepath{fill}
\begin{pgfscope}
\pgfsetdash{}{0cm}
\pgfsetlinewidth{0.818mm}
\pgfsetmiterlimit{7.0}
\pgfpathmoveto{\pgfqpoint{0.682cm}{0.671cm}}
\pgfpathlineto{\pgfqpoint{0.679cm}{1.418cm}}
\pgfusepath{stroke}
\end{pgfscope}
\pgfpathmoveto{\pgfqpoint{0.815cm}{1.399cm}}
\pgfpathcurveto{\pgfqpoint{0.815cm}{1.435cm}}{\pgfqpoint{0.801cm}{1.47cm}}{\pgfqpoint{0.775cm}{1.496cm}}
\pgfpathcurveto{\pgfqpoint{0.75cm}{1.521cm}}{\pgfqpoint{0.715cm}{1.536cm}}{\pgfqpoint{0.679cm}{1.536cm}}
\pgfpathcurveto{\pgfqpoint{0.643cm}{1.536cm}}{\pgfqpoint{0.608cm}{1.521cm}}{\pgfqpoint{0.582cm}{1.496cm}}
\pgfpathcurveto{\pgfqpoint{0.557cm}{1.47cm}}{\pgfqpoint{0.542cm}{1.435cm}}{\pgfqpoint{0.542cm}{1.399cm}}
\pgfpathcurveto{\pgfqpoint{0.542cm}{1.363cm}}{\pgfqpoint{0.557cm}{1.328cm}}{\pgfqpoint{0.582cm}{1.302cm}}
\pgfpathcurveto{\pgfqpoint{0.608cm}{1.276cm}}{\pgfqpoint{0.643cm}{1.262cm}}{\pgfqpoint{0.679cm}{1.262cm}}
\pgfpathcurveto{\pgfqpoint{0.715cm}{1.262cm}}{\pgfqpoint{0.75cm}{1.276cm}}{\pgfqpoint{0.775cm}{1.302cm}}
\pgfpathcurveto{\pgfqpoint{0.801cm}{1.328cm}}{\pgfqpoint{0.815cm}{1.363cm}}{\pgfqpoint{0.815cm}{1.399cm}}
\pgfusepath{fill}
\pgfpathmoveto{\pgfqpoint{1.345cm}{1.371cm}}
\pgfpathcurveto{\pgfqpoint{1.345cm}{1.408cm}}{\pgfqpoint{1.331cm}{1.442cm}}{\pgfqpoint{1.305cm}{1.468cm}}
\pgfpathcurveto{\pgfqpoint{1.28cm}{1.494cm}}{\pgfqpoint{1.245cm}{1.508cm}}{\pgfqpoint{1.209cm}{1.508cm}}
\pgfpathcurveto{\pgfqpoint{1.172cm}{1.508cm}}{\pgfqpoint{1.138cm}{1.494cm}}{\pgfqpoint{1.112cm}{1.468cm}}
\pgfpathcurveto{\pgfqpoint{1.087cm}{1.442cm}}{\pgfqpoint{1.072cm}{1.408cm}}{\pgfqpoint{1.072cm}{1.371cm}}
\pgfpathcurveto{\pgfqpoint{1.072cm}{1.335cm}}{\pgfqpoint{1.087cm}{1.3cm}}{\pgfqpoint{1.112cm}{1.274cm}}
\pgfpathcurveto{\pgfqpoint{1.138cm}{1.249cm}}{\pgfqpoint{1.172cm}{1.234cm}}{\pgfqpoint{1.209cm}{1.234cm}}
\pgfpathcurveto{\pgfqpoint{1.245cm}{1.234cm}}{\pgfqpoint{1.28cm}{1.249cm}}{\pgfqpoint{1.305cm}{1.274cm}}
\pgfpathcurveto{\pgfqpoint{1.331cm}{1.3cm}}{\pgfqpoint{1.345cm}{1.335cm}}{\pgfqpoint{1.345cm}{1.371cm}}
\pgfusepath{fill}
\begin{pgfscope}
\pgfsetdash{}{0cm}
\pgfsetlinewidth{0.818mm}
\pgfsetroundcap
\pgfsetmiterlimit{4.0}
\pgfpathmoveto{\pgfqpoint{0.682cm}{0.671cm}}
\pgfpathlineto{\pgfqpoint{0.682cm}{0.042cm}}
\pgfusepath{stroke}
\end{pgfscope}
\end{pgfscope}
\end{pgfscope}
\end{pgfscope}
\end{tikzpicture}}}$ as the stationary solution to
\begin{equation}\label{eq:Xt31}
 \LL_{\varepsilon} X^{\!\resizebox{0.6em}{!}{
\begin{tikzpicture}
\pgfpathmoveto{\pgfqpoint{0cm}{-0.035cm}}
\pgfpathlineto{\pgfqpoint{1.376cm}{-0.035cm}}
\pgfpathlineto{\pgfqpoint{1.376cm}{1.552cm}}
\pgfpathlineto{\pgfqpoint{0cm}{1.552cm}}
\pgfpathclose
\pgfusepath{clip}
\begin{pgfscope}
\begin{pgfscope}
\pgfpathmoveto{\pgfqpoint{0cm}{-0.035cm}}
\pgfpathlineto{\pgfqpoint{1.376cm}{-0.035cm}}
\pgfpathlineto{\pgfqpoint{1.376cm}{1.552cm}}
\pgfpathlineto{\pgfqpoint{0cm}{1.552cm}}
\pgfpathclose
\pgfusepath{clip}
\begin{pgfscope}
\begin{pgfscope}
\pgfsetdash{}{0cm}
\pgfsetlinewidth{0.818mm}
\pgfsetroundcap
\pgfsetroundjoin
\pgfsetmiterlimit{7.0}
\definecolor{eps2pgf_color}{gray}{0}\pgfsetstrokecolor{eps2pgf_color}\pgfsetfillcolor{eps2pgf_color}
\pgfpathmoveto{\pgfqpoint{0.117cm}{1.421cm}}
\pgfpathlineto{\pgfqpoint{0.682cm}{0.671cm}}
\pgfpathlineto{\pgfqpoint{1.246cm}{1.421cm}}
\pgfusepath{stroke}
\end{pgfscope}
\definecolor{eps2pgf_color}{gray}{0}\pgfsetstrokecolor{eps2pgf_color}\pgfsetfillcolor{eps2pgf_color}
\pgfpathmoveto{\pgfqpoint{0.273cm}{1.395cm}}
\pgfpathcurveto{\pgfqpoint{0.273cm}{1.432cm}}{\pgfqpoint{0.259cm}{1.467cm}}{\pgfqpoint{0.233cm}{1.492cm}}
\pgfpathcurveto{\pgfqpoint{0.207cm}{1.518cm}}{\pgfqpoint{0.173cm}{1.532cm}}{\pgfqpoint{0.137cm}{1.532cm}}
\pgfpathcurveto{\pgfqpoint{0.1cm}{1.532cm}}{\pgfqpoint{0.066cm}{1.518cm}}{\pgfqpoint{0.04cm}{1.492cm}}
\pgfpathcurveto{\pgfqpoint{0.014cm}{1.467cm}}{\pgfqpoint{0cm}{1.432cm}}{\pgfqpoint{0cm}{1.395cm}}
\pgfpathcurveto{\pgfqpoint{0cm}{1.359cm}}{\pgfqpoint{0.014cm}{1.324cm}}{\pgfqpoint{0.04cm}{1.299cm}}
\pgfpathcurveto{\pgfqpoint{0.066cm}{1.273cm}}{\pgfqpoint{0.1cm}{1.258cm}}{\pgfqpoint{0.137cm}{1.258cm}}
\pgfpathcurveto{\pgfqpoint{0.173cm}{1.258cm}}{\pgfqpoint{0.207cm}{1.273cm}}{\pgfqpoint{0.233cm}{1.299cm}}
\pgfpathcurveto{\pgfqpoint{0.259cm}{1.324cm}}{\pgfqpoint{0.273cm}{1.359cm}}{\pgfqpoint{0.273cm}{1.395cm}}
\pgfusepath{fill}
\begin{pgfscope}
\pgfsetdash{}{0cm}
\pgfsetlinewidth{0.818mm}
\pgfsetmiterlimit{7.0}
\pgfpathmoveto{\pgfqpoint{0.682cm}{0.671cm}}
\pgfpathlineto{\pgfqpoint{0.679cm}{1.418cm}}
\pgfusepath{stroke}
\end{pgfscope}
\pgfpathmoveto{\pgfqpoint{0.815cm}{1.399cm}}
\pgfpathcurveto{\pgfqpoint{0.815cm}{1.435cm}}{\pgfqpoint{0.801cm}{1.47cm}}{\pgfqpoint{0.775cm}{1.496cm}}
\pgfpathcurveto{\pgfqpoint{0.75cm}{1.521cm}}{\pgfqpoint{0.715cm}{1.536cm}}{\pgfqpoint{0.679cm}{1.536cm}}
\pgfpathcurveto{\pgfqpoint{0.643cm}{1.536cm}}{\pgfqpoint{0.608cm}{1.521cm}}{\pgfqpoint{0.582cm}{1.496cm}}
\pgfpathcurveto{\pgfqpoint{0.557cm}{1.47cm}}{\pgfqpoint{0.542cm}{1.435cm}}{\pgfqpoint{0.542cm}{1.399cm}}
\pgfpathcurveto{\pgfqpoint{0.542cm}{1.363cm}}{\pgfqpoint{0.557cm}{1.328cm}}{\pgfqpoint{0.582cm}{1.302cm}}
\pgfpathcurveto{\pgfqpoint{0.608cm}{1.276cm}}{\pgfqpoint{0.643cm}{1.262cm}}{\pgfqpoint{0.679cm}{1.262cm}}
\pgfpathcurveto{\pgfqpoint{0.715cm}{1.262cm}}{\pgfqpoint{0.75cm}{1.276cm}}{\pgfqpoint{0.775cm}{1.302cm}}
\pgfpathcurveto{\pgfqpoint{0.801cm}{1.328cm}}{\pgfqpoint{0.815cm}{1.363cm}}{\pgfqpoint{0.815cm}{1.399cm}}
\pgfusepath{fill}
\pgfpathmoveto{\pgfqpoint{1.345cm}{1.371cm}}
\pgfpathcurveto{\pgfqpoint{1.345cm}{1.408cm}}{\pgfqpoint{1.331cm}{1.442cm}}{\pgfqpoint{1.305cm}{1.468cm}}
\pgfpathcurveto{\pgfqpoint{1.28cm}{1.494cm}}{\pgfqpoint{1.245cm}{1.508cm}}{\pgfqpoint{1.209cm}{1.508cm}}
\pgfpathcurveto{\pgfqpoint{1.172cm}{1.508cm}}{\pgfqpoint{1.138cm}{1.494cm}}{\pgfqpoint{1.112cm}{1.468cm}}
\pgfpathcurveto{\pgfqpoint{1.087cm}{1.442cm}}{\pgfqpoint{1.072cm}{1.408cm}}{\pgfqpoint{1.072cm}{1.371cm}}
\pgfpathcurveto{\pgfqpoint{1.072cm}{1.335cm}}{\pgfqpoint{1.087cm}{1.3cm}}{\pgfqpoint{1.112cm}{1.274cm}}
\pgfpathcurveto{\pgfqpoint{1.138cm}{1.249cm}}{\pgfqpoint{1.172cm}{1.234cm}}{\pgfqpoint{1.209cm}{1.234cm}}
\pgfpathcurveto{\pgfqpoint{1.245cm}{1.234cm}}{\pgfqpoint{1.28cm}{1.249cm}}{\pgfqpoint{1.305cm}{1.274cm}}
\pgfpathcurveto{\pgfqpoint{1.331cm}{1.3cm}}{\pgfqpoint{1.345cm}{1.335cm}}{\pgfqpoint{1.345cm}{1.371cm}}
\pgfusepath{fill}
\begin{pgfscope}
\pgfsetdash{}{0cm}
\pgfsetlinewidth{0.818mm}
\pgfsetroundcap
\pgfsetmiterlimit{4.0}
\pgfpathmoveto{\pgfqpoint{0.682cm}{0.671cm}}
\pgfpathlineto{\pgfqpoint{0.682cm}{0.042cm}}
\pgfusepath{stroke}
\end{pgfscope}
\end{pgfscope}
\end{pgfscope}
\end{pgfscope}
\end{tikzpicture}}}_{M, \varepsilon} =
   \llbracket X^3_{M, \varepsilon} \rrbracket,
   \end{equation}
leading to the decomposition $\varphi_{M, \varepsilon} = X_{M, \varepsilon} -
\lambda X_{M, \varepsilon}^{\!\resizebox{0.6em}{!}{
\begin{tikzpicture}
\pgfpathmoveto{\pgfqpoint{0cm}{-0.035cm}}
\pgfpathlineto{\pgfqpoint{1.376cm}{-0.035cm}}
\pgfpathlineto{\pgfqpoint{1.376cm}{1.552cm}}
\pgfpathlineto{\pgfqpoint{0cm}{1.552cm}}
\pgfpathclose
\pgfusepath{clip}
\begin{pgfscope}
\begin{pgfscope}
\pgfpathmoveto{\pgfqpoint{0cm}{-0.035cm}}
\pgfpathlineto{\pgfqpoint{1.376cm}{-0.035cm}}
\pgfpathlineto{\pgfqpoint{1.376cm}{1.552cm}}
\pgfpathlineto{\pgfqpoint{0cm}{1.552cm}}
\pgfpathclose
\pgfusepath{clip}
\begin{pgfscope}
\begin{pgfscope}
\pgfsetdash{}{0cm}
\pgfsetlinewidth{0.818mm}
\pgfsetroundcap
\pgfsetroundjoin
\pgfsetmiterlimit{7.0}
\definecolor{eps2pgf_color}{gray}{0}\pgfsetstrokecolor{eps2pgf_color}\pgfsetfillcolor{eps2pgf_color}
\pgfpathmoveto{\pgfqpoint{0.117cm}{1.421cm}}
\pgfpathlineto{\pgfqpoint{0.682cm}{0.671cm}}
\pgfpathlineto{\pgfqpoint{1.246cm}{1.421cm}}
\pgfusepath{stroke}
\end{pgfscope}
\definecolor{eps2pgf_color}{gray}{0}\pgfsetstrokecolor{eps2pgf_color}\pgfsetfillcolor{eps2pgf_color}
\pgfpathmoveto{\pgfqpoint{0.273cm}{1.395cm}}
\pgfpathcurveto{\pgfqpoint{0.273cm}{1.432cm}}{\pgfqpoint{0.259cm}{1.467cm}}{\pgfqpoint{0.233cm}{1.492cm}}
\pgfpathcurveto{\pgfqpoint{0.207cm}{1.518cm}}{\pgfqpoint{0.173cm}{1.532cm}}{\pgfqpoint{0.137cm}{1.532cm}}
\pgfpathcurveto{\pgfqpoint{0.1cm}{1.532cm}}{\pgfqpoint{0.066cm}{1.518cm}}{\pgfqpoint{0.04cm}{1.492cm}}
\pgfpathcurveto{\pgfqpoint{0.014cm}{1.467cm}}{\pgfqpoint{0cm}{1.432cm}}{\pgfqpoint{0cm}{1.395cm}}
\pgfpathcurveto{\pgfqpoint{0cm}{1.359cm}}{\pgfqpoint{0.014cm}{1.324cm}}{\pgfqpoint{0.04cm}{1.299cm}}
\pgfpathcurveto{\pgfqpoint{0.066cm}{1.273cm}}{\pgfqpoint{0.1cm}{1.258cm}}{\pgfqpoint{0.137cm}{1.258cm}}
\pgfpathcurveto{\pgfqpoint{0.173cm}{1.258cm}}{\pgfqpoint{0.207cm}{1.273cm}}{\pgfqpoint{0.233cm}{1.299cm}}
\pgfpathcurveto{\pgfqpoint{0.259cm}{1.324cm}}{\pgfqpoint{0.273cm}{1.359cm}}{\pgfqpoint{0.273cm}{1.395cm}}
\pgfusepath{fill}
\begin{pgfscope}
\pgfsetdash{}{0cm}
\pgfsetlinewidth{0.818mm}
\pgfsetmiterlimit{7.0}
\pgfpathmoveto{\pgfqpoint{0.682cm}{0.671cm}}
\pgfpathlineto{\pgfqpoint{0.679cm}{1.418cm}}
\pgfusepath{stroke}
\end{pgfscope}
\pgfpathmoveto{\pgfqpoint{0.815cm}{1.399cm}}
\pgfpathcurveto{\pgfqpoint{0.815cm}{1.435cm}}{\pgfqpoint{0.801cm}{1.47cm}}{\pgfqpoint{0.775cm}{1.496cm}}
\pgfpathcurveto{\pgfqpoint{0.75cm}{1.521cm}}{\pgfqpoint{0.715cm}{1.536cm}}{\pgfqpoint{0.679cm}{1.536cm}}
\pgfpathcurveto{\pgfqpoint{0.643cm}{1.536cm}}{\pgfqpoint{0.608cm}{1.521cm}}{\pgfqpoint{0.582cm}{1.496cm}}
\pgfpathcurveto{\pgfqpoint{0.557cm}{1.47cm}}{\pgfqpoint{0.542cm}{1.435cm}}{\pgfqpoint{0.542cm}{1.399cm}}
\pgfpathcurveto{\pgfqpoint{0.542cm}{1.363cm}}{\pgfqpoint{0.557cm}{1.328cm}}{\pgfqpoint{0.582cm}{1.302cm}}
\pgfpathcurveto{\pgfqpoint{0.608cm}{1.276cm}}{\pgfqpoint{0.643cm}{1.262cm}}{\pgfqpoint{0.679cm}{1.262cm}}
\pgfpathcurveto{\pgfqpoint{0.715cm}{1.262cm}}{\pgfqpoint{0.75cm}{1.276cm}}{\pgfqpoint{0.775cm}{1.302cm}}
\pgfpathcurveto{\pgfqpoint{0.801cm}{1.328cm}}{\pgfqpoint{0.815cm}{1.363cm}}{\pgfqpoint{0.815cm}{1.399cm}}
\pgfusepath{fill}
\pgfpathmoveto{\pgfqpoint{1.345cm}{1.371cm}}
\pgfpathcurveto{\pgfqpoint{1.345cm}{1.408cm}}{\pgfqpoint{1.331cm}{1.442cm}}{\pgfqpoint{1.305cm}{1.468cm}}
\pgfpathcurveto{\pgfqpoint{1.28cm}{1.494cm}}{\pgfqpoint{1.245cm}{1.508cm}}{\pgfqpoint{1.209cm}{1.508cm}}
\pgfpathcurveto{\pgfqpoint{1.172cm}{1.508cm}}{\pgfqpoint{1.138cm}{1.494cm}}{\pgfqpoint{1.112cm}{1.468cm}}
\pgfpathcurveto{\pgfqpoint{1.087cm}{1.442cm}}{\pgfqpoint{1.072cm}{1.408cm}}{\pgfqpoint{1.072cm}{1.371cm}}
\pgfpathcurveto{\pgfqpoint{1.072cm}{1.335cm}}{\pgfqpoint{1.087cm}{1.3cm}}{\pgfqpoint{1.112cm}{1.274cm}}
\pgfpathcurveto{\pgfqpoint{1.138cm}{1.249cm}}{\pgfqpoint{1.172cm}{1.234cm}}{\pgfqpoint{1.209cm}{1.234cm}}
\pgfpathcurveto{\pgfqpoint{1.245cm}{1.234cm}}{\pgfqpoint{1.28cm}{1.249cm}}{\pgfqpoint{1.305cm}{1.274cm}}
\pgfpathcurveto{\pgfqpoint{1.331cm}{1.3cm}}{\pgfqpoint{1.345cm}{1.335cm}}{\pgfqpoint{1.345cm}{1.371cm}}
\pgfusepath{fill}
\begin{pgfscope}
\pgfsetdash{}{0cm}
\pgfsetlinewidth{0.818mm}
\pgfsetroundcap
\pgfsetmiterlimit{4.0}
\pgfpathmoveto{\pgfqpoint{0.682cm}{0.671cm}}
\pgfpathlineto{\pgfqpoint{0.682cm}{0.042cm}}
\pgfusepath{stroke}
\end{pgfscope}
\end{pgfscope}
\end{pgfscope}
\end{pgfscope}
\end{tikzpicture}}} + \zeta_{M, \varepsilon}$. Writing down
the dynamics for $\zeta_{M, \varepsilon}$ we observe that the most irregular
term is the paraproduct $\llbracket X_{M, \varepsilon}^2 \rrbracket \succ
X^{\!\resizebox{0.6em}{!}{
\begin{tikzpicture}
\pgfpathmoveto{\pgfqpoint{0cm}{-0.035cm}}
\pgfpathlineto{\pgfqpoint{1.376cm}{-0.035cm}}
\pgfpathlineto{\pgfqpoint{1.376cm}{1.552cm}}
\pgfpathlineto{\pgfqpoint{0cm}{1.552cm}}
\pgfpathclose
\pgfusepath{clip}
\begin{pgfscope}
\begin{pgfscope}
\pgfpathmoveto{\pgfqpoint{0cm}{-0.035cm}}
\pgfpathlineto{\pgfqpoint{1.376cm}{-0.035cm}}
\pgfpathlineto{\pgfqpoint{1.376cm}{1.552cm}}
\pgfpathlineto{\pgfqpoint{0cm}{1.552cm}}
\pgfpathclose
\pgfusepath{clip}
\begin{pgfscope}
\begin{pgfscope}
\pgfsetdash{}{0cm}
\pgfsetlinewidth{0.818mm}
\pgfsetroundcap
\pgfsetroundjoin
\pgfsetmiterlimit{7.0}
\definecolor{eps2pgf_color}{gray}{0}\pgfsetstrokecolor{eps2pgf_color}\pgfsetfillcolor{eps2pgf_color}
\pgfpathmoveto{\pgfqpoint{0.117cm}{1.421cm}}
\pgfpathlineto{\pgfqpoint{0.682cm}{0.671cm}}
\pgfpathlineto{\pgfqpoint{1.246cm}{1.421cm}}
\pgfusepath{stroke}
\end{pgfscope}
\definecolor{eps2pgf_color}{gray}{0}\pgfsetstrokecolor{eps2pgf_color}\pgfsetfillcolor{eps2pgf_color}
\pgfpathmoveto{\pgfqpoint{0.273cm}{1.395cm}}
\pgfpathcurveto{\pgfqpoint{0.273cm}{1.432cm}}{\pgfqpoint{0.259cm}{1.467cm}}{\pgfqpoint{0.233cm}{1.492cm}}
\pgfpathcurveto{\pgfqpoint{0.207cm}{1.518cm}}{\pgfqpoint{0.173cm}{1.532cm}}{\pgfqpoint{0.137cm}{1.532cm}}
\pgfpathcurveto{\pgfqpoint{0.1cm}{1.532cm}}{\pgfqpoint{0.066cm}{1.518cm}}{\pgfqpoint{0.04cm}{1.492cm}}
\pgfpathcurveto{\pgfqpoint{0.014cm}{1.467cm}}{\pgfqpoint{0cm}{1.432cm}}{\pgfqpoint{0cm}{1.395cm}}
\pgfpathcurveto{\pgfqpoint{0cm}{1.359cm}}{\pgfqpoint{0.014cm}{1.324cm}}{\pgfqpoint{0.04cm}{1.299cm}}
\pgfpathcurveto{\pgfqpoint{0.066cm}{1.273cm}}{\pgfqpoint{0.1cm}{1.258cm}}{\pgfqpoint{0.137cm}{1.258cm}}
\pgfpathcurveto{\pgfqpoint{0.173cm}{1.258cm}}{\pgfqpoint{0.207cm}{1.273cm}}{\pgfqpoint{0.233cm}{1.299cm}}
\pgfpathcurveto{\pgfqpoint{0.259cm}{1.324cm}}{\pgfqpoint{0.273cm}{1.359cm}}{\pgfqpoint{0.273cm}{1.395cm}}
\pgfusepath{fill}
\begin{pgfscope}
\pgfsetdash{}{0cm}
\pgfsetlinewidth{0.818mm}
\pgfsetmiterlimit{7.0}
\pgfpathmoveto{\pgfqpoint{0.682cm}{0.671cm}}
\pgfpathlineto{\pgfqpoint{0.679cm}{1.418cm}}
\pgfusepath{stroke}
\end{pgfscope}
\pgfpathmoveto{\pgfqpoint{0.815cm}{1.399cm}}
\pgfpathcurveto{\pgfqpoint{0.815cm}{1.435cm}}{\pgfqpoint{0.801cm}{1.47cm}}{\pgfqpoint{0.775cm}{1.496cm}}
\pgfpathcurveto{\pgfqpoint{0.75cm}{1.521cm}}{\pgfqpoint{0.715cm}{1.536cm}}{\pgfqpoint{0.679cm}{1.536cm}}
\pgfpathcurveto{\pgfqpoint{0.643cm}{1.536cm}}{\pgfqpoint{0.608cm}{1.521cm}}{\pgfqpoint{0.582cm}{1.496cm}}
\pgfpathcurveto{\pgfqpoint{0.557cm}{1.47cm}}{\pgfqpoint{0.542cm}{1.435cm}}{\pgfqpoint{0.542cm}{1.399cm}}
\pgfpathcurveto{\pgfqpoint{0.542cm}{1.363cm}}{\pgfqpoint{0.557cm}{1.328cm}}{\pgfqpoint{0.582cm}{1.302cm}}
\pgfpathcurveto{\pgfqpoint{0.608cm}{1.276cm}}{\pgfqpoint{0.643cm}{1.262cm}}{\pgfqpoint{0.679cm}{1.262cm}}
\pgfpathcurveto{\pgfqpoint{0.715cm}{1.262cm}}{\pgfqpoint{0.75cm}{1.276cm}}{\pgfqpoint{0.775cm}{1.302cm}}
\pgfpathcurveto{\pgfqpoint{0.801cm}{1.328cm}}{\pgfqpoint{0.815cm}{1.363cm}}{\pgfqpoint{0.815cm}{1.399cm}}
\pgfusepath{fill}
\pgfpathmoveto{\pgfqpoint{1.345cm}{1.371cm}}
\pgfpathcurveto{\pgfqpoint{1.345cm}{1.408cm}}{\pgfqpoint{1.331cm}{1.442cm}}{\pgfqpoint{1.305cm}{1.468cm}}
\pgfpathcurveto{\pgfqpoint{1.28cm}{1.494cm}}{\pgfqpoint{1.245cm}{1.508cm}}{\pgfqpoint{1.209cm}{1.508cm}}
\pgfpathcurveto{\pgfqpoint{1.172cm}{1.508cm}}{\pgfqpoint{1.138cm}{1.494cm}}{\pgfqpoint{1.112cm}{1.468cm}}
\pgfpathcurveto{\pgfqpoint{1.087cm}{1.442cm}}{\pgfqpoint{1.072cm}{1.408cm}}{\pgfqpoint{1.072cm}{1.371cm}}
\pgfpathcurveto{\pgfqpoint{1.072cm}{1.335cm}}{\pgfqpoint{1.087cm}{1.3cm}}{\pgfqpoint{1.112cm}{1.274cm}}
\pgfpathcurveto{\pgfqpoint{1.138cm}{1.249cm}}{\pgfqpoint{1.172cm}{1.234cm}}{\pgfqpoint{1.209cm}{1.234cm}}
\pgfpathcurveto{\pgfqpoint{1.245cm}{1.234cm}}{\pgfqpoint{1.28cm}{1.249cm}}{\pgfqpoint{1.305cm}{1.274cm}}
\pgfpathcurveto{\pgfqpoint{1.331cm}{1.3cm}}{\pgfqpoint{1.345cm}{1.335cm}}{\pgfqpoint{1.345cm}{1.371cm}}
\pgfusepath{fill}
\begin{pgfscope}
\pgfsetdash{}{0cm}
\pgfsetlinewidth{0.818mm}
\pgfsetroundcap
\pgfsetmiterlimit{4.0}
\pgfpathmoveto{\pgfqpoint{0.682cm}{0.671cm}}
\pgfpathlineto{\pgfqpoint{0.682cm}{0.042cm}}
\pgfusepath{stroke}
\end{pgfscope}
\end{pgfscope}
\end{pgfscope}
\end{pgfscope}
\end{tikzpicture}}}_{M, \varepsilon}$ which can be bounded uniformly in $\mathscr{C}
\hspace{.1em}^{- 1 - \kappa,\varepsilon} (\rho^{\sigma})$ and hence this is not yet
sufficient for the energy method outlined above. Indeed,  the expected (uniform)
regularity of $\zeta_{M, \varepsilon}$ is $\mathscr{C} \hspace{.1em}^{1 -
\kappa,\varepsilon} (\rho^{\sigma})$ and so the term $\langle \zeta_{M,\varepsilon},\llbracket X_{M, \varepsilon}^2 \rrbracket \succ
X^{\!\resizebox{0.6em}{!}{
\begin{tikzpicture}
\pgfpathmoveto{\pgfqpoint{0cm}{-0.035cm}}
\pgfpathlineto{\pgfqpoint{1.376cm}{-0.035cm}}
\pgfpathlineto{\pgfqpoint{1.376cm}{1.552cm}}
\pgfpathlineto{\pgfqpoint{0cm}{1.552cm}}
\pgfpathclose
\pgfusepath{clip}
\begin{pgfscope}
\begin{pgfscope}
\pgfpathmoveto{\pgfqpoint{0cm}{-0.035cm}}
\pgfpathlineto{\pgfqpoint{1.376cm}{-0.035cm}}
\pgfpathlineto{\pgfqpoint{1.376cm}{1.552cm}}
\pgfpathlineto{\pgfqpoint{0cm}{1.552cm}}
\pgfpathclose
\pgfusepath{clip}
\begin{pgfscope}
\begin{pgfscope}
\pgfsetdash{}{0cm}
\pgfsetlinewidth{0.818mm}
\pgfsetroundcap
\pgfsetroundjoin
\pgfsetmiterlimit{7.0}
\definecolor{eps2pgf_color}{gray}{0}\pgfsetstrokecolor{eps2pgf_color}\pgfsetfillcolor{eps2pgf_color}
\pgfpathmoveto{\pgfqpoint{0.117cm}{1.421cm}}
\pgfpathlineto{\pgfqpoint{0.682cm}{0.671cm}}
\pgfpathlineto{\pgfqpoint{1.246cm}{1.421cm}}
\pgfusepath{stroke}
\end{pgfscope}
\definecolor{eps2pgf_color}{gray}{0}\pgfsetstrokecolor{eps2pgf_color}\pgfsetfillcolor{eps2pgf_color}
\pgfpathmoveto{\pgfqpoint{0.273cm}{1.395cm}}
\pgfpathcurveto{\pgfqpoint{0.273cm}{1.432cm}}{\pgfqpoint{0.259cm}{1.467cm}}{\pgfqpoint{0.233cm}{1.492cm}}
\pgfpathcurveto{\pgfqpoint{0.207cm}{1.518cm}}{\pgfqpoint{0.173cm}{1.532cm}}{\pgfqpoint{0.137cm}{1.532cm}}
\pgfpathcurveto{\pgfqpoint{0.1cm}{1.532cm}}{\pgfqpoint{0.066cm}{1.518cm}}{\pgfqpoint{0.04cm}{1.492cm}}
\pgfpathcurveto{\pgfqpoint{0.014cm}{1.467cm}}{\pgfqpoint{0cm}{1.432cm}}{\pgfqpoint{0cm}{1.395cm}}
\pgfpathcurveto{\pgfqpoint{0cm}{1.359cm}}{\pgfqpoint{0.014cm}{1.324cm}}{\pgfqpoint{0.04cm}{1.299cm}}
\pgfpathcurveto{\pgfqpoint{0.066cm}{1.273cm}}{\pgfqpoint{0.1cm}{1.258cm}}{\pgfqpoint{0.137cm}{1.258cm}}
\pgfpathcurveto{\pgfqpoint{0.173cm}{1.258cm}}{\pgfqpoint{0.207cm}{1.273cm}}{\pgfqpoint{0.233cm}{1.299cm}}
\pgfpathcurveto{\pgfqpoint{0.259cm}{1.324cm}}{\pgfqpoint{0.273cm}{1.359cm}}{\pgfqpoint{0.273cm}{1.395cm}}
\pgfusepath{fill}
\begin{pgfscope}
\pgfsetdash{}{0cm}
\pgfsetlinewidth{0.818mm}
\pgfsetmiterlimit{7.0}
\pgfpathmoveto{\pgfqpoint{0.682cm}{0.671cm}}
\pgfpathlineto{\pgfqpoint{0.679cm}{1.418cm}}
\pgfusepath{stroke}
\end{pgfscope}
\pgfpathmoveto{\pgfqpoint{0.815cm}{1.399cm}}
\pgfpathcurveto{\pgfqpoint{0.815cm}{1.435cm}}{\pgfqpoint{0.801cm}{1.47cm}}{\pgfqpoint{0.775cm}{1.496cm}}
\pgfpathcurveto{\pgfqpoint{0.75cm}{1.521cm}}{\pgfqpoint{0.715cm}{1.536cm}}{\pgfqpoint{0.679cm}{1.536cm}}
\pgfpathcurveto{\pgfqpoint{0.643cm}{1.536cm}}{\pgfqpoint{0.608cm}{1.521cm}}{\pgfqpoint{0.582cm}{1.496cm}}
\pgfpathcurveto{\pgfqpoint{0.557cm}{1.47cm}}{\pgfqpoint{0.542cm}{1.435cm}}{\pgfqpoint{0.542cm}{1.399cm}}
\pgfpathcurveto{\pgfqpoint{0.542cm}{1.363cm}}{\pgfqpoint{0.557cm}{1.328cm}}{\pgfqpoint{0.582cm}{1.302cm}}
\pgfpathcurveto{\pgfqpoint{0.608cm}{1.276cm}}{\pgfqpoint{0.643cm}{1.262cm}}{\pgfqpoint{0.679cm}{1.262cm}}
\pgfpathcurveto{\pgfqpoint{0.715cm}{1.262cm}}{\pgfqpoint{0.75cm}{1.276cm}}{\pgfqpoint{0.775cm}{1.302cm}}
\pgfpathcurveto{\pgfqpoint{0.801cm}{1.328cm}}{\pgfqpoint{0.815cm}{1.363cm}}{\pgfqpoint{0.815cm}{1.399cm}}
\pgfusepath{fill}
\pgfpathmoveto{\pgfqpoint{1.345cm}{1.371cm}}
\pgfpathcurveto{\pgfqpoint{1.345cm}{1.408cm}}{\pgfqpoint{1.331cm}{1.442cm}}{\pgfqpoint{1.305cm}{1.468cm}}
\pgfpathcurveto{\pgfqpoint{1.28cm}{1.494cm}}{\pgfqpoint{1.245cm}{1.508cm}}{\pgfqpoint{1.209cm}{1.508cm}}
\pgfpathcurveto{\pgfqpoint{1.172cm}{1.508cm}}{\pgfqpoint{1.138cm}{1.494cm}}{\pgfqpoint{1.112cm}{1.468cm}}
\pgfpathcurveto{\pgfqpoint{1.087cm}{1.442cm}}{\pgfqpoint{1.072cm}{1.408cm}}{\pgfqpoint{1.072cm}{1.371cm}}
\pgfpathcurveto{\pgfqpoint{1.072cm}{1.335cm}}{\pgfqpoint{1.087cm}{1.3cm}}{\pgfqpoint{1.112cm}{1.274cm}}
\pgfpathcurveto{\pgfqpoint{1.138cm}{1.249cm}}{\pgfqpoint{1.172cm}{1.234cm}}{\pgfqpoint{1.209cm}{1.234cm}}
\pgfpathcurveto{\pgfqpoint{1.245cm}{1.234cm}}{\pgfqpoint{1.28cm}{1.249cm}}{\pgfqpoint{1.305cm}{1.274cm}}
\pgfpathcurveto{\pgfqpoint{1.331cm}{1.3cm}}{\pgfqpoint{1.345cm}{1.335cm}}{\pgfqpoint{1.345cm}{1.371cm}}
\pgfusepath{fill}
\begin{pgfscope}
\pgfsetdash{}{0cm}
\pgfsetlinewidth{0.818mm}
\pgfsetroundcap
\pgfsetmiterlimit{4.0}
\pgfpathmoveto{\pgfqpoint{0.682cm}{0.671cm}}
\pgfpathlineto{\pgfqpoint{0.682cm}{0.042cm}}
\pgfusepath{stroke}
\end{pgfscope}
\end{pgfscope}
\end{pgfscope}
\end{pgfscope}
\end{tikzpicture}}}_{M, \varepsilon}\rangle$ cannot be controlled.
However, we point out that not much is missing.

In order to overcome this issue, we proceed differently than the above cited
works and let $Y_{M, \varepsilon}$ be a solution to
\begin{equation}
 \LL_{\varepsilon} Y_{M, \varepsilon} = - \llbracket
  X_{M, \varepsilon}^3 \rrbracket - 3 \lambda (\mathscr{U}^{\varepsilon}_{>}
  \llbracket X_{M, \varepsilon}^2 \rrbracket) \succ Y_{M, \varepsilon}, \qquad
  Y_{M, \varepsilon} (0) = - \lambda X^{\!\resizebox{0.6em}{!}{
\begin{tikzpicture}
\pgfpathmoveto{\pgfqpoint{0cm}{-0.035cm}}
\pgfpathlineto{\pgfqpoint{1.376cm}{-0.035cm}}
\pgfpathlineto{\pgfqpoint{1.376cm}{1.552cm}}
\pgfpathlineto{\pgfqpoint{0cm}{1.552cm}}
\pgfpathclose
\pgfusepath{clip}
\begin{pgfscope}
\begin{pgfscope}
\pgfpathmoveto{\pgfqpoint{0cm}{-0.035cm}}
\pgfpathlineto{\pgfqpoint{1.376cm}{-0.035cm}}
\pgfpathlineto{\pgfqpoint{1.376cm}{1.552cm}}
\pgfpathlineto{\pgfqpoint{0cm}{1.552cm}}
\pgfpathclose
\pgfusepath{clip}
\begin{pgfscope}
\begin{pgfscope}
\pgfsetdash{}{0cm}
\pgfsetlinewidth{0.818mm}
\pgfsetroundcap
\pgfsetroundjoin
\pgfsetmiterlimit{7.0}
\definecolor{eps2pgf_color}{gray}{0}\pgfsetstrokecolor{eps2pgf_color}\pgfsetfillcolor{eps2pgf_color}
\pgfpathmoveto{\pgfqpoint{0.117cm}{1.421cm}}
\pgfpathlineto{\pgfqpoint{0.682cm}{0.671cm}}
\pgfpathlineto{\pgfqpoint{1.246cm}{1.421cm}}
\pgfusepath{stroke}
\end{pgfscope}
\definecolor{eps2pgf_color}{gray}{0}\pgfsetstrokecolor{eps2pgf_color}\pgfsetfillcolor{eps2pgf_color}
\pgfpathmoveto{\pgfqpoint{0.273cm}{1.395cm}}
\pgfpathcurveto{\pgfqpoint{0.273cm}{1.432cm}}{\pgfqpoint{0.259cm}{1.467cm}}{\pgfqpoint{0.233cm}{1.492cm}}
\pgfpathcurveto{\pgfqpoint{0.207cm}{1.518cm}}{\pgfqpoint{0.173cm}{1.532cm}}{\pgfqpoint{0.137cm}{1.532cm}}
\pgfpathcurveto{\pgfqpoint{0.1cm}{1.532cm}}{\pgfqpoint{0.066cm}{1.518cm}}{\pgfqpoint{0.04cm}{1.492cm}}
\pgfpathcurveto{\pgfqpoint{0.014cm}{1.467cm}}{\pgfqpoint{0cm}{1.432cm}}{\pgfqpoint{0cm}{1.395cm}}
\pgfpathcurveto{\pgfqpoint{0cm}{1.359cm}}{\pgfqpoint{0.014cm}{1.324cm}}{\pgfqpoint{0.04cm}{1.299cm}}
\pgfpathcurveto{\pgfqpoint{0.066cm}{1.273cm}}{\pgfqpoint{0.1cm}{1.258cm}}{\pgfqpoint{0.137cm}{1.258cm}}
\pgfpathcurveto{\pgfqpoint{0.173cm}{1.258cm}}{\pgfqpoint{0.207cm}{1.273cm}}{\pgfqpoint{0.233cm}{1.299cm}}
\pgfpathcurveto{\pgfqpoint{0.259cm}{1.324cm}}{\pgfqpoint{0.273cm}{1.359cm}}{\pgfqpoint{0.273cm}{1.395cm}}
\pgfusepath{fill}
\begin{pgfscope}
\pgfsetdash{}{0cm}
\pgfsetlinewidth{0.818mm}
\pgfsetmiterlimit{7.0}
\pgfpathmoveto{\pgfqpoint{0.682cm}{0.671cm}}
\pgfpathlineto{\pgfqpoint{0.679cm}{1.418cm}}
\pgfusepath{stroke}
\end{pgfscope}
\pgfpathmoveto{\pgfqpoint{0.815cm}{1.399cm}}
\pgfpathcurveto{\pgfqpoint{0.815cm}{1.435cm}}{\pgfqpoint{0.801cm}{1.47cm}}{\pgfqpoint{0.775cm}{1.496cm}}
\pgfpathcurveto{\pgfqpoint{0.75cm}{1.521cm}}{\pgfqpoint{0.715cm}{1.536cm}}{\pgfqpoint{0.679cm}{1.536cm}}
\pgfpathcurveto{\pgfqpoint{0.643cm}{1.536cm}}{\pgfqpoint{0.608cm}{1.521cm}}{\pgfqpoint{0.582cm}{1.496cm}}
\pgfpathcurveto{\pgfqpoint{0.557cm}{1.47cm}}{\pgfqpoint{0.542cm}{1.435cm}}{\pgfqpoint{0.542cm}{1.399cm}}
\pgfpathcurveto{\pgfqpoint{0.542cm}{1.363cm}}{\pgfqpoint{0.557cm}{1.328cm}}{\pgfqpoint{0.582cm}{1.302cm}}
\pgfpathcurveto{\pgfqpoint{0.608cm}{1.276cm}}{\pgfqpoint{0.643cm}{1.262cm}}{\pgfqpoint{0.679cm}{1.262cm}}
\pgfpathcurveto{\pgfqpoint{0.715cm}{1.262cm}}{\pgfqpoint{0.75cm}{1.276cm}}{\pgfqpoint{0.775cm}{1.302cm}}
\pgfpathcurveto{\pgfqpoint{0.801cm}{1.328cm}}{\pgfqpoint{0.815cm}{1.363cm}}{\pgfqpoint{0.815cm}{1.399cm}}
\pgfusepath{fill}
\pgfpathmoveto{\pgfqpoint{1.345cm}{1.371cm}}
\pgfpathcurveto{\pgfqpoint{1.345cm}{1.408cm}}{\pgfqpoint{1.331cm}{1.442cm}}{\pgfqpoint{1.305cm}{1.468cm}}
\pgfpathcurveto{\pgfqpoint{1.28cm}{1.494cm}}{\pgfqpoint{1.245cm}{1.508cm}}{\pgfqpoint{1.209cm}{1.508cm}}
\pgfpathcurveto{\pgfqpoint{1.172cm}{1.508cm}}{\pgfqpoint{1.138cm}{1.494cm}}{\pgfqpoint{1.112cm}{1.468cm}}
\pgfpathcurveto{\pgfqpoint{1.087cm}{1.442cm}}{\pgfqpoint{1.072cm}{1.408cm}}{\pgfqpoint{1.072cm}{1.371cm}}
\pgfpathcurveto{\pgfqpoint{1.072cm}{1.335cm}}{\pgfqpoint{1.087cm}{1.3cm}}{\pgfqpoint{1.112cm}{1.274cm}}
\pgfpathcurveto{\pgfqpoint{1.138cm}{1.249cm}}{\pgfqpoint{1.172cm}{1.234cm}}{\pgfqpoint{1.209cm}{1.234cm}}
\pgfpathcurveto{\pgfqpoint{1.245cm}{1.234cm}}{\pgfqpoint{1.28cm}{1.249cm}}{\pgfqpoint{1.305cm}{1.274cm}}
\pgfpathcurveto{\pgfqpoint{1.331cm}{1.3cm}}{\pgfqpoint{1.345cm}{1.335cm}}{\pgfqpoint{1.345cm}{1.371cm}}
\pgfusepath{fill}
\begin{pgfscope}
\pgfsetdash{}{0cm}
\pgfsetlinewidth{0.818mm}
\pgfsetroundcap
\pgfsetmiterlimit{4.0}
\pgfpathmoveto{\pgfqpoint{0.682cm}{0.671cm}}
\pgfpathlineto{\pgfqpoint{0.682cm}{0.042cm}}
\pgfusepath{stroke}
\end{pgfscope}
\end{pgfscope}
\end{pgfscope}
\end{pgfscope}
\end{tikzpicture}}}_{M, \varepsilon}(0),
  \label{eq:Y1}
\end{equation}
where $\mathscr{U}^{\varepsilon}_{>}$ is the localization operator defined in
Section \ref{s:l1}. With a suitable choice of the constant $L = L (\lambda, M,
\varepsilon)$ determining $\mathscr{U}^{\varepsilon}_{>}$ (cf. Lemma
\ref{lem:loc}, Lemma \ref{lem:Y1}) we are able to construct the unique solution to this problem
via Banach's fixed point theorem. Consequently, we find our decomposition
$\varphi_{M, \varepsilon} = X_{M, \varepsilon} + Y_{M, \varepsilon} + \phi_{M,
\varepsilon}$ together with the dynamics for the remainder
\begin{equation}
  \LL_{\varepsilon}\phi_{M, \varepsilon} + \lambda
  \phi_{M, \varepsilon}^3 = - 3 \lambda \llbracket X_{M, \varepsilon}^2
  \rrbracket \succ \phi_{M, \varepsilon} - 3 \lambda \llbracket X_{M,
  \varepsilon}^2 \rrbracket \circ \phi_{M, \varepsilon} - 3 \lambda^2 b_{M,
  \varepsilon} \phi_{M, \varepsilon} + \Xi_{M, \varepsilon} . \label{eq:ph}
\end{equation}
The first term on the right hand side is the most irregular contribution, the
second term is not controlled uniformly in $M, \varepsilon$, the third term is
needed for the renormalization and $\Xi_{M, \varepsilon}$ contains various
terms that are more regular and in principle not problematic or that can be
constructed as stochastic objects using the remaining counterterm $- 3
\lambda^2 b_{M, \varepsilon} (X_{M, \varepsilon} + Y_{M, \varepsilon})$.

The advantage of this decomposition with $\phi_{M, \varepsilon}$ as opposed to
the usual approach leading to $\zeta_{M, \varepsilon}$ above is that together
with $\llbracket X^3_{M, \varepsilon} \rrbracket$ we cancelled also the second
most irregular contribution $(\mathscr{U}^{\varepsilon}_{>} \llbracket X_{M,
\varepsilon}^2 \rrbracket) \succ Y_{M, \varepsilon}$, which is too irregular
to be controlled as a forcing $f$ using the energy method. The same difficulty
of course comes with $\llbracket X_{M, \varepsilon}^2 \rrbracket \succ
\phi_{M, \varepsilon}$ in {\eqref{eq:ph}}, however, since it depends on the
solution $\phi_{M, \varepsilon}$ we are able to control it using a
paracontrolled ansatz. To explain this, let us also turn our attention to the
resonant product $\llbracket X_{M, \varepsilon}^2 \rrbracket \circ \phi_{M,
\varepsilon}$ which poses problems as well. When applying the energy method to
{\eqref{eq:ph}}, these two terms appear in the form
\[ \langle \rho^4 \phi_{M, \varepsilon}, - 3 \lambda \llbracket X_{M,
   \varepsilon}^2 \rrbracket \circ \phi_{M, \varepsilon} \rangle_{\varepsilon}
   + \langle \rho^4 \phi_{M, \varepsilon}, - 3 \lambda \llbracket X_{M,
   \varepsilon}^2 \rrbracket \succ \phi_{M, \varepsilon}
   \rangle_{\varepsilon}, \]
where we included a polynomial weight $\rho$ as in \eqref{eq:weight}. The key observation is that the
presence of the duality product permits to show that these two terms
{\em{approximately}} coincide, in the sense that their difference denoted by
$D_{\rho^4, \varepsilon} (\phi_{M, \varepsilon}, - 3 \lambda \llbracket
X^2_{M, \varepsilon} \rrbracket, \phi_{M, \varepsilon})$ is controlled by the
expected uniform bounds. This is proven generally in Lemma \ref{lem:dual1}. As a
consequence, we obtain
\[ \frac{1}{2} \partial_t \| \phi_{M, \varepsilon} \|_{L^{2, \varepsilon}}^2 +
   \lambda \| \phi_{M, \varepsilon} \|_{L^{4, \varepsilon}}^4 + \langle
   \phi_{M, \varepsilon}, \Q_{\varepsilon}\phi_{M, \varepsilon}
   \rangle_{\varepsilon} \]
\[ = \langle \rho^4 \phi_{M, \varepsilon}, - 3 \cdummy 2 \lambda \llbracket
   X_{M, \varepsilon}^2 \rrbracket \succ \phi_{M, \varepsilon}
   \rangle_{\varepsilon} + D_{\rho^4, \varepsilon} (\phi_{M, \varepsilon}, - 3
   \lambda \llbracket X^2_{M, \varepsilon} \rrbracket, \phi_{M, \varepsilon})
   + \Xi_{M, \varepsilon} . \]

Finally, since the last term on the left hand side as well as the first term
on the right hand side are diverging, the idea is to couple them by the
following paracontrolled ansatz. We define
\[\Q_{\varepsilon} \psi_{M, \varepsilon} \assign \Q_{\varepsilon} \phi_{M, \varepsilon} + 3 \llbracket X_{M,
   \varepsilon}^2 \rrbracket \succ \phi_{M, \varepsilon} \]
and expect that the sum of the two terms on the right hand side is more
regular than each of them separately. In other words, $\psi_{M,
\varepsilon}$ is (uniformly) more regular than $\phi_{M, \varepsilon}$.
Indeed, with this ansatz we may complete the square and  obtain
\[ \frac{1}{2} \partial_t \| \rho^2 \phi_{M, \varepsilon} \|_{L^{2,
   \varepsilon}}^2 + \lambda \| \rho \phi_{M, \varepsilon} \|_{L^{4,
   \varepsilon}}^4 + m^2 \| \rho^2 \psi_{M, \varepsilon} \|_{L^{2,
   \varepsilon}}^2 + \| \rho^2 \nabla_{\varepsilon} \psi_{M, \varepsilon}
   \|_{L^{2, \varepsilon}}^2 = \Theta_{\rho^4, M, \varepsilon} + \Psi_{\rho^4,
   M, \varepsilon}, \]
where the right hand side, given in Lemma \ref{lem:energy12}, can be controlled
by the norms on the left hand side, in the spirit of the energy method 
discussed above.

These considerations lead  to our first main result proved as Theorem
\ref{th:energy-estimate} below. In what follows,   $Q_{\rho}(\mathbb{X}_{M,\varepsilon})$ denotes a polynomial in the $\rho$-weighted norms of the involved stochastic objects, the precise definition can be found in Section \ref{ssec:stoch}.

\begin{theorem}
  \label{th:energy-estimate-int}Let $\rho$ be a weight such that $\rho^{\iota}
  \in L^{4, 0}$ for some $\iota \in (0, 1)$. There exists a constant $\alpha =
  \alpha (m^2) > 0$ such that
  \[ \frac{1}{2} \partial_t \| \rho^2 \phi_{M, \varepsilon} \|_{L^{2,
     \varepsilon}}^2 + \alpha [\lambda \| \rho \phi_{M, \varepsilon} \|_{L^{4,
     \varepsilon}}^4 + m^2 \| \rho^2 \psi_{M, \varepsilon} \|_{L^{2,
     \varepsilon}}^2 + \| \rho^2 \nabla_{\varepsilon} \psi_{M, \varepsilon}
     \|_{L^{2, \varepsilon}}^2] + \| \rho^2 \phi_{M, \varepsilon} \|_{H^{1 - 2
     \kappa, \varepsilon}}^2 \]
  \[ \leqslant C_{\lambda, t} Q_{\rho} (\mathbb{X}_{M, \varepsilon}), \]
  where $C_{\lambda, t} = \lambda^3 + \lambda^{(12 - \theta) / (2 + \theta)} |
  \log t |^{4 / (2 + \theta)} + \lambda^7$ for $\theta = \frac{1 / 2 - 4
  \kappa}{1 - 2 \kappa}$.
\end{theorem}

Here we observe the precise dependence on $\lambda$ which in particular
implies that the bound is uniform over $\lambda$ in every bounded subset of
$[0, \infty)$ and vanishes as $\lambda \rightarrow 0$.

\paragraph{Tightness.} In order to proceed to the proof of the existence of the
Euclidean $\Phi^4_3$ field theory, we shall employ the extension operator
$\mathcal{E}^{\varepsilon}$ from Section~\ref{s:ext} which permits to extend
discrete distributions to the full space $\mathbb{R}^3$. An additional twist
originates in the fact that by construction the process $Y_{M, \varepsilon}$
given by {\eqref{eq:Y1}} is not stationary and consequently also $\phi_{M,
\varepsilon}$ fails to be stationary. Therefore the energy argument as
explained above does not apply as it stands and we shall go back to the
stationary decomposition $\varphi_{M, \varepsilon} = X_{M, \varepsilon} -
\lambda X_{M, \varepsilon}^{\!\resizebox{0.6em}{!}{
\begin{tikzpicture}
\pgfpathmoveto{\pgfqpoint{0cm}{-0.035cm}}
\pgfpathlineto{\pgfqpoint{1.376cm}{-0.035cm}}
\pgfpathlineto{\pgfqpoint{1.376cm}{1.552cm}}
\pgfpathlineto{\pgfqpoint{0cm}{1.552cm}}
\pgfpathclose
\pgfusepath{clip}
\begin{pgfscope}
\begin{pgfscope}
\pgfpathmoveto{\pgfqpoint{0cm}{-0.035cm}}
\pgfpathlineto{\pgfqpoint{1.376cm}{-0.035cm}}
\pgfpathlineto{\pgfqpoint{1.376cm}{1.552cm}}
\pgfpathlineto{\pgfqpoint{0cm}{1.552cm}}
\pgfpathclose
\pgfusepath{clip}
\begin{pgfscope}
\begin{pgfscope}
\pgfsetdash{}{0cm}
\pgfsetlinewidth{0.818mm}
\pgfsetroundcap
\pgfsetroundjoin
\pgfsetmiterlimit{7.0}
\definecolor{eps2pgf_color}{gray}{0}\pgfsetstrokecolor{eps2pgf_color}\pgfsetfillcolor{eps2pgf_color}
\pgfpathmoveto{\pgfqpoint{0.117cm}{1.421cm}}
\pgfpathlineto{\pgfqpoint{0.682cm}{0.671cm}}
\pgfpathlineto{\pgfqpoint{1.246cm}{1.421cm}}
\pgfusepath{stroke}
\end{pgfscope}
\definecolor{eps2pgf_color}{gray}{0}\pgfsetstrokecolor{eps2pgf_color}\pgfsetfillcolor{eps2pgf_color}
\pgfpathmoveto{\pgfqpoint{0.273cm}{1.395cm}}
\pgfpathcurveto{\pgfqpoint{0.273cm}{1.432cm}}{\pgfqpoint{0.259cm}{1.467cm}}{\pgfqpoint{0.233cm}{1.492cm}}
\pgfpathcurveto{\pgfqpoint{0.207cm}{1.518cm}}{\pgfqpoint{0.173cm}{1.532cm}}{\pgfqpoint{0.137cm}{1.532cm}}
\pgfpathcurveto{\pgfqpoint{0.1cm}{1.532cm}}{\pgfqpoint{0.066cm}{1.518cm}}{\pgfqpoint{0.04cm}{1.492cm}}
\pgfpathcurveto{\pgfqpoint{0.014cm}{1.467cm}}{\pgfqpoint{0cm}{1.432cm}}{\pgfqpoint{0cm}{1.395cm}}
\pgfpathcurveto{\pgfqpoint{0cm}{1.359cm}}{\pgfqpoint{0.014cm}{1.324cm}}{\pgfqpoint{0.04cm}{1.299cm}}
\pgfpathcurveto{\pgfqpoint{0.066cm}{1.273cm}}{\pgfqpoint{0.1cm}{1.258cm}}{\pgfqpoint{0.137cm}{1.258cm}}
\pgfpathcurveto{\pgfqpoint{0.173cm}{1.258cm}}{\pgfqpoint{0.207cm}{1.273cm}}{\pgfqpoint{0.233cm}{1.299cm}}
\pgfpathcurveto{\pgfqpoint{0.259cm}{1.324cm}}{\pgfqpoint{0.273cm}{1.359cm}}{\pgfqpoint{0.273cm}{1.395cm}}
\pgfusepath{fill}
\begin{pgfscope}
\pgfsetdash{}{0cm}
\pgfsetlinewidth{0.818mm}
\pgfsetmiterlimit{7.0}
\pgfpathmoveto{\pgfqpoint{0.682cm}{0.671cm}}
\pgfpathlineto{\pgfqpoint{0.679cm}{1.418cm}}
\pgfusepath{stroke}
\end{pgfscope}
\pgfpathmoveto{\pgfqpoint{0.815cm}{1.399cm}}
\pgfpathcurveto{\pgfqpoint{0.815cm}{1.435cm}}{\pgfqpoint{0.801cm}{1.47cm}}{\pgfqpoint{0.775cm}{1.496cm}}
\pgfpathcurveto{\pgfqpoint{0.75cm}{1.521cm}}{\pgfqpoint{0.715cm}{1.536cm}}{\pgfqpoint{0.679cm}{1.536cm}}
\pgfpathcurveto{\pgfqpoint{0.643cm}{1.536cm}}{\pgfqpoint{0.608cm}{1.521cm}}{\pgfqpoint{0.582cm}{1.496cm}}
\pgfpathcurveto{\pgfqpoint{0.557cm}{1.47cm}}{\pgfqpoint{0.542cm}{1.435cm}}{\pgfqpoint{0.542cm}{1.399cm}}
\pgfpathcurveto{\pgfqpoint{0.542cm}{1.363cm}}{\pgfqpoint{0.557cm}{1.328cm}}{\pgfqpoint{0.582cm}{1.302cm}}
\pgfpathcurveto{\pgfqpoint{0.608cm}{1.276cm}}{\pgfqpoint{0.643cm}{1.262cm}}{\pgfqpoint{0.679cm}{1.262cm}}
\pgfpathcurveto{\pgfqpoint{0.715cm}{1.262cm}}{\pgfqpoint{0.75cm}{1.276cm}}{\pgfqpoint{0.775cm}{1.302cm}}
\pgfpathcurveto{\pgfqpoint{0.801cm}{1.328cm}}{\pgfqpoint{0.815cm}{1.363cm}}{\pgfqpoint{0.815cm}{1.399cm}}
\pgfusepath{fill}
\pgfpathmoveto{\pgfqpoint{1.345cm}{1.371cm}}
\pgfpathcurveto{\pgfqpoint{1.345cm}{1.408cm}}{\pgfqpoint{1.331cm}{1.442cm}}{\pgfqpoint{1.305cm}{1.468cm}}
\pgfpathcurveto{\pgfqpoint{1.28cm}{1.494cm}}{\pgfqpoint{1.245cm}{1.508cm}}{\pgfqpoint{1.209cm}{1.508cm}}
\pgfpathcurveto{\pgfqpoint{1.172cm}{1.508cm}}{\pgfqpoint{1.138cm}{1.494cm}}{\pgfqpoint{1.112cm}{1.468cm}}
\pgfpathcurveto{\pgfqpoint{1.087cm}{1.442cm}}{\pgfqpoint{1.072cm}{1.408cm}}{\pgfqpoint{1.072cm}{1.371cm}}
\pgfpathcurveto{\pgfqpoint{1.072cm}{1.335cm}}{\pgfqpoint{1.087cm}{1.3cm}}{\pgfqpoint{1.112cm}{1.274cm}}
\pgfpathcurveto{\pgfqpoint{1.138cm}{1.249cm}}{\pgfqpoint{1.172cm}{1.234cm}}{\pgfqpoint{1.209cm}{1.234cm}}
\pgfpathcurveto{\pgfqpoint{1.245cm}{1.234cm}}{\pgfqpoint{1.28cm}{1.249cm}}{\pgfqpoint{1.305cm}{1.274cm}}
\pgfpathcurveto{\pgfqpoint{1.331cm}{1.3cm}}{\pgfqpoint{1.345cm}{1.335cm}}{\pgfqpoint{1.345cm}{1.371cm}}
\pgfusepath{fill}
\begin{pgfscope}
\pgfsetdash{}{0cm}
\pgfsetlinewidth{0.818mm}
\pgfsetroundcap
\pgfsetmiterlimit{4.0}
\pgfpathmoveto{\pgfqpoint{0.682cm}{0.671cm}}
\pgfpathlineto{\pgfqpoint{0.682cm}{0.042cm}}
\pgfusepath{stroke}
\end{pgfscope}
\end{pgfscope}
\end{pgfscope}
\end{pgfscope}
\end{tikzpicture}}} + \zeta_{M, \varepsilon}$, while using
the result of Theorem \ref{th:energy-estimate-int} in order to estimate
$\zeta_{M, \varepsilon}$. Consequently, we deduce tightness of the family of
the joint laws of $( \mathcal{E}^{\varepsilon}\varphi_{M, \varepsilon},\mathcal{E}^{\varepsilon} X_{M, \varepsilon},
\mathcal{E}^{\varepsilon}X^{\!\resizebox{0.6em}{!}{
\begin{tikzpicture}
\pgfpathmoveto{\pgfqpoint{0cm}{-0.035cm}}
\pgfpathlineto{\pgfqpoint{1.376cm}{-0.035cm}}
\pgfpathlineto{\pgfqpoint{1.376cm}{1.552cm}}
\pgfpathlineto{\pgfqpoint{0cm}{1.552cm}}
\pgfpathclose
\pgfusepath{clip}
\begin{pgfscope}
\begin{pgfscope}
\pgfpathmoveto{\pgfqpoint{0cm}{-0.035cm}}
\pgfpathlineto{\pgfqpoint{1.376cm}{-0.035cm}}
\pgfpathlineto{\pgfqpoint{1.376cm}{1.552cm}}
\pgfpathlineto{\pgfqpoint{0cm}{1.552cm}}
\pgfpathclose
\pgfusepath{clip}
\begin{pgfscope}
\begin{pgfscope}
\pgfsetdash{}{0cm}
\pgfsetlinewidth{0.818mm}
\pgfsetroundcap
\pgfsetroundjoin
\pgfsetmiterlimit{7.0}
\definecolor{eps2pgf_color}{gray}{0}\pgfsetstrokecolor{eps2pgf_color}\pgfsetfillcolor{eps2pgf_color}
\pgfpathmoveto{\pgfqpoint{0.117cm}{1.421cm}}
\pgfpathlineto{\pgfqpoint{0.682cm}{0.671cm}}
\pgfpathlineto{\pgfqpoint{1.246cm}{1.421cm}}
\pgfusepath{stroke}
\end{pgfscope}
\definecolor{eps2pgf_color}{gray}{0}\pgfsetstrokecolor{eps2pgf_color}\pgfsetfillcolor{eps2pgf_color}
\pgfpathmoveto{\pgfqpoint{0.273cm}{1.395cm}}
\pgfpathcurveto{\pgfqpoint{0.273cm}{1.432cm}}{\pgfqpoint{0.259cm}{1.467cm}}{\pgfqpoint{0.233cm}{1.492cm}}
\pgfpathcurveto{\pgfqpoint{0.207cm}{1.518cm}}{\pgfqpoint{0.173cm}{1.532cm}}{\pgfqpoint{0.137cm}{1.532cm}}
\pgfpathcurveto{\pgfqpoint{0.1cm}{1.532cm}}{\pgfqpoint{0.066cm}{1.518cm}}{\pgfqpoint{0.04cm}{1.492cm}}
\pgfpathcurveto{\pgfqpoint{0.014cm}{1.467cm}}{\pgfqpoint{0cm}{1.432cm}}{\pgfqpoint{0cm}{1.395cm}}
\pgfpathcurveto{\pgfqpoint{0cm}{1.359cm}}{\pgfqpoint{0.014cm}{1.324cm}}{\pgfqpoint{0.04cm}{1.299cm}}
\pgfpathcurveto{\pgfqpoint{0.066cm}{1.273cm}}{\pgfqpoint{0.1cm}{1.258cm}}{\pgfqpoint{0.137cm}{1.258cm}}
\pgfpathcurveto{\pgfqpoint{0.173cm}{1.258cm}}{\pgfqpoint{0.207cm}{1.273cm}}{\pgfqpoint{0.233cm}{1.299cm}}
\pgfpathcurveto{\pgfqpoint{0.259cm}{1.324cm}}{\pgfqpoint{0.273cm}{1.359cm}}{\pgfqpoint{0.273cm}{1.395cm}}
\pgfusepath{fill}
\begin{pgfscope}
\pgfsetdash{}{0cm}
\pgfsetlinewidth{0.818mm}
\pgfsetmiterlimit{7.0}
\pgfpathmoveto{\pgfqpoint{0.682cm}{0.671cm}}
\pgfpathlineto{\pgfqpoint{0.679cm}{1.418cm}}
\pgfusepath{stroke}
\end{pgfscope}
\pgfpathmoveto{\pgfqpoint{0.815cm}{1.399cm}}
\pgfpathcurveto{\pgfqpoint{0.815cm}{1.435cm}}{\pgfqpoint{0.801cm}{1.47cm}}{\pgfqpoint{0.775cm}{1.496cm}}
\pgfpathcurveto{\pgfqpoint{0.75cm}{1.521cm}}{\pgfqpoint{0.715cm}{1.536cm}}{\pgfqpoint{0.679cm}{1.536cm}}
\pgfpathcurveto{\pgfqpoint{0.643cm}{1.536cm}}{\pgfqpoint{0.608cm}{1.521cm}}{\pgfqpoint{0.582cm}{1.496cm}}
\pgfpathcurveto{\pgfqpoint{0.557cm}{1.47cm}}{\pgfqpoint{0.542cm}{1.435cm}}{\pgfqpoint{0.542cm}{1.399cm}}
\pgfpathcurveto{\pgfqpoint{0.542cm}{1.363cm}}{\pgfqpoint{0.557cm}{1.328cm}}{\pgfqpoint{0.582cm}{1.302cm}}
\pgfpathcurveto{\pgfqpoint{0.608cm}{1.276cm}}{\pgfqpoint{0.643cm}{1.262cm}}{\pgfqpoint{0.679cm}{1.262cm}}
\pgfpathcurveto{\pgfqpoint{0.715cm}{1.262cm}}{\pgfqpoint{0.75cm}{1.276cm}}{\pgfqpoint{0.775cm}{1.302cm}}
\pgfpathcurveto{\pgfqpoint{0.801cm}{1.328cm}}{\pgfqpoint{0.815cm}{1.363cm}}{\pgfqpoint{0.815cm}{1.399cm}}
\pgfusepath{fill}
\pgfpathmoveto{\pgfqpoint{1.345cm}{1.371cm}}
\pgfpathcurveto{\pgfqpoint{1.345cm}{1.408cm}}{\pgfqpoint{1.331cm}{1.442cm}}{\pgfqpoint{1.305cm}{1.468cm}}
\pgfpathcurveto{\pgfqpoint{1.28cm}{1.494cm}}{\pgfqpoint{1.245cm}{1.508cm}}{\pgfqpoint{1.209cm}{1.508cm}}
\pgfpathcurveto{\pgfqpoint{1.172cm}{1.508cm}}{\pgfqpoint{1.138cm}{1.494cm}}{\pgfqpoint{1.112cm}{1.468cm}}
\pgfpathcurveto{\pgfqpoint{1.087cm}{1.442cm}}{\pgfqpoint{1.072cm}{1.408cm}}{\pgfqpoint{1.072cm}{1.371cm}}
\pgfpathcurveto{\pgfqpoint{1.072cm}{1.335cm}}{\pgfqpoint{1.087cm}{1.3cm}}{\pgfqpoint{1.112cm}{1.274cm}}
\pgfpathcurveto{\pgfqpoint{1.138cm}{1.249cm}}{\pgfqpoint{1.172cm}{1.234cm}}{\pgfqpoint{1.209cm}{1.234cm}}
\pgfpathcurveto{\pgfqpoint{1.245cm}{1.234cm}}{\pgfqpoint{1.28cm}{1.249cm}}{\pgfqpoint{1.305cm}{1.274cm}}
\pgfpathcurveto{\pgfqpoint{1.331cm}{1.3cm}}{\pgfqpoint{1.345cm}{1.335cm}}{\pgfqpoint{1.345cm}{1.371cm}}
\pgfusepath{fill}
\begin{pgfscope}
\pgfsetdash{}{0cm}
\pgfsetlinewidth{0.818mm}
\pgfsetroundcap
\pgfsetmiterlimit{4.0}
\pgfpathmoveto{\pgfqpoint{0.682cm}{0.671cm}}
\pgfpathlineto{\pgfqpoint{0.682cm}{0.042cm}}
\pgfusepath{stroke}
\end{pgfscope}
\end{pgfscope}
\end{pgfscope}
\end{pgfscope}
\end{tikzpicture}}}_{M, \varepsilon} )$ evaluated at any fixed time $t
\geqslant 0$, proven in Theorem \ref{thm:main} below. To this end, we denote
by $(\varphi, X, X^{\!\resizebox{0.6em}{!}{
\begin{tikzpicture}
\pgfpathmoveto{\pgfqpoint{0cm}{-0.035cm}}
\pgfpathlineto{\pgfqpoint{1.376cm}{-0.035cm}}
\pgfpathlineto{\pgfqpoint{1.376cm}{1.552cm}}
\pgfpathlineto{\pgfqpoint{0cm}{1.552cm}}
\pgfpathclose
\pgfusepath{clip}
\begin{pgfscope}
\begin{pgfscope}
\pgfpathmoveto{\pgfqpoint{0cm}{-0.035cm}}
\pgfpathlineto{\pgfqpoint{1.376cm}{-0.035cm}}
\pgfpathlineto{\pgfqpoint{1.376cm}{1.552cm}}
\pgfpathlineto{\pgfqpoint{0cm}{1.552cm}}
\pgfpathclose
\pgfusepath{clip}
\begin{pgfscope}
\begin{pgfscope}
\pgfsetdash{}{0cm}
\pgfsetlinewidth{0.818mm}
\pgfsetroundcap
\pgfsetroundjoin
\pgfsetmiterlimit{7.0}
\definecolor{eps2pgf_color}{gray}{0}\pgfsetstrokecolor{eps2pgf_color}\pgfsetfillcolor{eps2pgf_color}
\pgfpathmoveto{\pgfqpoint{0.117cm}{1.421cm}}
\pgfpathlineto{\pgfqpoint{0.682cm}{0.671cm}}
\pgfpathlineto{\pgfqpoint{1.246cm}{1.421cm}}
\pgfusepath{stroke}
\end{pgfscope}
\definecolor{eps2pgf_color}{gray}{0}\pgfsetstrokecolor{eps2pgf_color}\pgfsetfillcolor{eps2pgf_color}
\pgfpathmoveto{\pgfqpoint{0.273cm}{1.395cm}}
\pgfpathcurveto{\pgfqpoint{0.273cm}{1.432cm}}{\pgfqpoint{0.259cm}{1.467cm}}{\pgfqpoint{0.233cm}{1.492cm}}
\pgfpathcurveto{\pgfqpoint{0.207cm}{1.518cm}}{\pgfqpoint{0.173cm}{1.532cm}}{\pgfqpoint{0.137cm}{1.532cm}}
\pgfpathcurveto{\pgfqpoint{0.1cm}{1.532cm}}{\pgfqpoint{0.066cm}{1.518cm}}{\pgfqpoint{0.04cm}{1.492cm}}
\pgfpathcurveto{\pgfqpoint{0.014cm}{1.467cm}}{\pgfqpoint{0cm}{1.432cm}}{\pgfqpoint{0cm}{1.395cm}}
\pgfpathcurveto{\pgfqpoint{0cm}{1.359cm}}{\pgfqpoint{0.014cm}{1.324cm}}{\pgfqpoint{0.04cm}{1.299cm}}
\pgfpathcurveto{\pgfqpoint{0.066cm}{1.273cm}}{\pgfqpoint{0.1cm}{1.258cm}}{\pgfqpoint{0.137cm}{1.258cm}}
\pgfpathcurveto{\pgfqpoint{0.173cm}{1.258cm}}{\pgfqpoint{0.207cm}{1.273cm}}{\pgfqpoint{0.233cm}{1.299cm}}
\pgfpathcurveto{\pgfqpoint{0.259cm}{1.324cm}}{\pgfqpoint{0.273cm}{1.359cm}}{\pgfqpoint{0.273cm}{1.395cm}}
\pgfusepath{fill}
\begin{pgfscope}
\pgfsetdash{}{0cm}
\pgfsetlinewidth{0.818mm}
\pgfsetmiterlimit{7.0}
\pgfpathmoveto{\pgfqpoint{0.682cm}{0.671cm}}
\pgfpathlineto{\pgfqpoint{0.679cm}{1.418cm}}
\pgfusepath{stroke}
\end{pgfscope}
\pgfpathmoveto{\pgfqpoint{0.815cm}{1.399cm}}
\pgfpathcurveto{\pgfqpoint{0.815cm}{1.435cm}}{\pgfqpoint{0.801cm}{1.47cm}}{\pgfqpoint{0.775cm}{1.496cm}}
\pgfpathcurveto{\pgfqpoint{0.75cm}{1.521cm}}{\pgfqpoint{0.715cm}{1.536cm}}{\pgfqpoint{0.679cm}{1.536cm}}
\pgfpathcurveto{\pgfqpoint{0.643cm}{1.536cm}}{\pgfqpoint{0.608cm}{1.521cm}}{\pgfqpoint{0.582cm}{1.496cm}}
\pgfpathcurveto{\pgfqpoint{0.557cm}{1.47cm}}{\pgfqpoint{0.542cm}{1.435cm}}{\pgfqpoint{0.542cm}{1.399cm}}
\pgfpathcurveto{\pgfqpoint{0.542cm}{1.363cm}}{\pgfqpoint{0.557cm}{1.328cm}}{\pgfqpoint{0.582cm}{1.302cm}}
\pgfpathcurveto{\pgfqpoint{0.608cm}{1.276cm}}{\pgfqpoint{0.643cm}{1.262cm}}{\pgfqpoint{0.679cm}{1.262cm}}
\pgfpathcurveto{\pgfqpoint{0.715cm}{1.262cm}}{\pgfqpoint{0.75cm}{1.276cm}}{\pgfqpoint{0.775cm}{1.302cm}}
\pgfpathcurveto{\pgfqpoint{0.801cm}{1.328cm}}{\pgfqpoint{0.815cm}{1.363cm}}{\pgfqpoint{0.815cm}{1.399cm}}
\pgfusepath{fill}
\pgfpathmoveto{\pgfqpoint{1.345cm}{1.371cm}}
\pgfpathcurveto{\pgfqpoint{1.345cm}{1.408cm}}{\pgfqpoint{1.331cm}{1.442cm}}{\pgfqpoint{1.305cm}{1.468cm}}
\pgfpathcurveto{\pgfqpoint{1.28cm}{1.494cm}}{\pgfqpoint{1.245cm}{1.508cm}}{\pgfqpoint{1.209cm}{1.508cm}}
\pgfpathcurveto{\pgfqpoint{1.172cm}{1.508cm}}{\pgfqpoint{1.138cm}{1.494cm}}{\pgfqpoint{1.112cm}{1.468cm}}
\pgfpathcurveto{\pgfqpoint{1.087cm}{1.442cm}}{\pgfqpoint{1.072cm}{1.408cm}}{\pgfqpoint{1.072cm}{1.371cm}}
\pgfpathcurveto{\pgfqpoint{1.072cm}{1.335cm}}{\pgfqpoint{1.087cm}{1.3cm}}{\pgfqpoint{1.112cm}{1.274cm}}
\pgfpathcurveto{\pgfqpoint{1.138cm}{1.249cm}}{\pgfqpoint{1.172cm}{1.234cm}}{\pgfqpoint{1.209cm}{1.234cm}}
\pgfpathcurveto{\pgfqpoint{1.245cm}{1.234cm}}{\pgfqpoint{1.28cm}{1.249cm}}{\pgfqpoint{1.305cm}{1.274cm}}
\pgfpathcurveto{\pgfqpoint{1.331cm}{1.3cm}}{\pgfqpoint{1.345cm}{1.335cm}}{\pgfqpoint{1.345cm}{1.371cm}}
\pgfusepath{fill}
\begin{pgfscope}
\pgfsetdash{}{0cm}
\pgfsetlinewidth{0.818mm}
\pgfsetroundcap
\pgfsetmiterlimit{4.0}
\pgfpathmoveto{\pgfqpoint{0.682cm}{0.671cm}}
\pgfpathlineto{\pgfqpoint{0.682cm}{0.042cm}}
\pgfusepath{stroke}
\end{pgfscope}
\end{pgfscope}
\end{pgfscope}
\end{pgfscope}
\end{tikzpicture}}})$ a canonical representative of the random
variables under consideration and let $\zeta \assign \varphi- X + \lambda
X^{\!\resizebox{0.6em}{!}{
\begin{tikzpicture}
\pgfpathmoveto{\pgfqpoint{0cm}{-0.035cm}}
\pgfpathlineto{\pgfqpoint{1.376cm}{-0.035cm}}
\pgfpathlineto{\pgfqpoint{1.376cm}{1.552cm}}
\pgfpathlineto{\pgfqpoint{0cm}{1.552cm}}
\pgfpathclose
\pgfusepath{clip}
\begin{pgfscope}
\begin{pgfscope}
\pgfpathmoveto{\pgfqpoint{0cm}{-0.035cm}}
\pgfpathlineto{\pgfqpoint{1.376cm}{-0.035cm}}
\pgfpathlineto{\pgfqpoint{1.376cm}{1.552cm}}
\pgfpathlineto{\pgfqpoint{0cm}{1.552cm}}
\pgfpathclose
\pgfusepath{clip}
\begin{pgfscope}
\begin{pgfscope}
\pgfsetdash{}{0cm}
\pgfsetlinewidth{0.818mm}
\pgfsetroundcap
\pgfsetroundjoin
\pgfsetmiterlimit{7.0}
\definecolor{eps2pgf_color}{gray}{0}\pgfsetstrokecolor{eps2pgf_color}\pgfsetfillcolor{eps2pgf_color}
\pgfpathmoveto{\pgfqpoint{0.117cm}{1.421cm}}
\pgfpathlineto{\pgfqpoint{0.682cm}{0.671cm}}
\pgfpathlineto{\pgfqpoint{1.246cm}{1.421cm}}
\pgfusepath{stroke}
\end{pgfscope}
\definecolor{eps2pgf_color}{gray}{0}\pgfsetstrokecolor{eps2pgf_color}\pgfsetfillcolor{eps2pgf_color}
\pgfpathmoveto{\pgfqpoint{0.273cm}{1.395cm}}
\pgfpathcurveto{\pgfqpoint{0.273cm}{1.432cm}}{\pgfqpoint{0.259cm}{1.467cm}}{\pgfqpoint{0.233cm}{1.492cm}}
\pgfpathcurveto{\pgfqpoint{0.207cm}{1.518cm}}{\pgfqpoint{0.173cm}{1.532cm}}{\pgfqpoint{0.137cm}{1.532cm}}
\pgfpathcurveto{\pgfqpoint{0.1cm}{1.532cm}}{\pgfqpoint{0.066cm}{1.518cm}}{\pgfqpoint{0.04cm}{1.492cm}}
\pgfpathcurveto{\pgfqpoint{0.014cm}{1.467cm}}{\pgfqpoint{0cm}{1.432cm}}{\pgfqpoint{0cm}{1.395cm}}
\pgfpathcurveto{\pgfqpoint{0cm}{1.359cm}}{\pgfqpoint{0.014cm}{1.324cm}}{\pgfqpoint{0.04cm}{1.299cm}}
\pgfpathcurveto{\pgfqpoint{0.066cm}{1.273cm}}{\pgfqpoint{0.1cm}{1.258cm}}{\pgfqpoint{0.137cm}{1.258cm}}
\pgfpathcurveto{\pgfqpoint{0.173cm}{1.258cm}}{\pgfqpoint{0.207cm}{1.273cm}}{\pgfqpoint{0.233cm}{1.299cm}}
\pgfpathcurveto{\pgfqpoint{0.259cm}{1.324cm}}{\pgfqpoint{0.273cm}{1.359cm}}{\pgfqpoint{0.273cm}{1.395cm}}
\pgfusepath{fill}
\begin{pgfscope}
\pgfsetdash{}{0cm}
\pgfsetlinewidth{0.818mm}
\pgfsetmiterlimit{7.0}
\pgfpathmoveto{\pgfqpoint{0.682cm}{0.671cm}}
\pgfpathlineto{\pgfqpoint{0.679cm}{1.418cm}}
\pgfusepath{stroke}
\end{pgfscope}
\pgfpathmoveto{\pgfqpoint{0.815cm}{1.399cm}}
\pgfpathcurveto{\pgfqpoint{0.815cm}{1.435cm}}{\pgfqpoint{0.801cm}{1.47cm}}{\pgfqpoint{0.775cm}{1.496cm}}
\pgfpathcurveto{\pgfqpoint{0.75cm}{1.521cm}}{\pgfqpoint{0.715cm}{1.536cm}}{\pgfqpoint{0.679cm}{1.536cm}}
\pgfpathcurveto{\pgfqpoint{0.643cm}{1.536cm}}{\pgfqpoint{0.608cm}{1.521cm}}{\pgfqpoint{0.582cm}{1.496cm}}
\pgfpathcurveto{\pgfqpoint{0.557cm}{1.47cm}}{\pgfqpoint{0.542cm}{1.435cm}}{\pgfqpoint{0.542cm}{1.399cm}}
\pgfpathcurveto{\pgfqpoint{0.542cm}{1.363cm}}{\pgfqpoint{0.557cm}{1.328cm}}{\pgfqpoint{0.582cm}{1.302cm}}
\pgfpathcurveto{\pgfqpoint{0.608cm}{1.276cm}}{\pgfqpoint{0.643cm}{1.262cm}}{\pgfqpoint{0.679cm}{1.262cm}}
\pgfpathcurveto{\pgfqpoint{0.715cm}{1.262cm}}{\pgfqpoint{0.75cm}{1.276cm}}{\pgfqpoint{0.775cm}{1.302cm}}
\pgfpathcurveto{\pgfqpoint{0.801cm}{1.328cm}}{\pgfqpoint{0.815cm}{1.363cm}}{\pgfqpoint{0.815cm}{1.399cm}}
\pgfusepath{fill}
\pgfpathmoveto{\pgfqpoint{1.345cm}{1.371cm}}
\pgfpathcurveto{\pgfqpoint{1.345cm}{1.408cm}}{\pgfqpoint{1.331cm}{1.442cm}}{\pgfqpoint{1.305cm}{1.468cm}}
\pgfpathcurveto{\pgfqpoint{1.28cm}{1.494cm}}{\pgfqpoint{1.245cm}{1.508cm}}{\pgfqpoint{1.209cm}{1.508cm}}
\pgfpathcurveto{\pgfqpoint{1.172cm}{1.508cm}}{\pgfqpoint{1.138cm}{1.494cm}}{\pgfqpoint{1.112cm}{1.468cm}}
\pgfpathcurveto{\pgfqpoint{1.087cm}{1.442cm}}{\pgfqpoint{1.072cm}{1.408cm}}{\pgfqpoint{1.072cm}{1.371cm}}
\pgfpathcurveto{\pgfqpoint{1.072cm}{1.335cm}}{\pgfqpoint{1.087cm}{1.3cm}}{\pgfqpoint{1.112cm}{1.274cm}}
\pgfpathcurveto{\pgfqpoint{1.138cm}{1.249cm}}{\pgfqpoint{1.172cm}{1.234cm}}{\pgfqpoint{1.209cm}{1.234cm}}
\pgfpathcurveto{\pgfqpoint{1.245cm}{1.234cm}}{\pgfqpoint{1.28cm}{1.249cm}}{\pgfqpoint{1.305cm}{1.274cm}}
\pgfpathcurveto{\pgfqpoint{1.331cm}{1.3cm}}{\pgfqpoint{1.345cm}{1.335cm}}{\pgfqpoint{1.345cm}{1.371cm}}
\pgfusepath{fill}
\begin{pgfscope}
\pgfsetdash{}{0cm}
\pgfsetlinewidth{0.818mm}
\pgfsetroundcap
\pgfsetmiterlimit{4.0}
\pgfpathmoveto{\pgfqpoint{0.682cm}{0.671cm}}
\pgfpathlineto{\pgfqpoint{0.682cm}{0.042cm}}
\pgfusepath{stroke}
\end{pgfscope}
\end{pgfscope}
\end{pgfscope}
\end{pgfscope}
\end{tikzpicture}}}$.

\begin{theorem}
  \label{thm:main-int}Let $\rho$ be a weight such that $\rho^{\iota} \in L^{4,
  0}$ for some $\iota \in (0, 1)$. Then the family of joint laws of $(
  \mathcal{E}^{\varepsilon} \varphi_{M, \varepsilon},
  \mathcal{E}^{\varepsilon} X_{M, \varepsilon}, \mathcal{E}^{\varepsilon}
  X^{\!\resizebox{0.6em}{!}{
\begin{tikzpicture}
\pgfpathmoveto{\pgfqpoint{0cm}{-0.035cm}}
\pgfpathlineto{\pgfqpoint{1.376cm}{-0.035cm}}
\pgfpathlineto{\pgfqpoint{1.376cm}{1.552cm}}
\pgfpathlineto{\pgfqpoint{0cm}{1.552cm}}
\pgfpathclose
\pgfusepath{clip}
\begin{pgfscope}
\begin{pgfscope}
\pgfpathmoveto{\pgfqpoint{0cm}{-0.035cm}}
\pgfpathlineto{\pgfqpoint{1.376cm}{-0.035cm}}
\pgfpathlineto{\pgfqpoint{1.376cm}{1.552cm}}
\pgfpathlineto{\pgfqpoint{0cm}{1.552cm}}
\pgfpathclose
\pgfusepath{clip}
\begin{pgfscope}
\begin{pgfscope}
\pgfsetdash{}{0cm}
\pgfsetlinewidth{0.818mm}
\pgfsetroundcap
\pgfsetroundjoin
\pgfsetmiterlimit{7.0}
\definecolor{eps2pgf_color}{gray}{0}\pgfsetstrokecolor{eps2pgf_color}\pgfsetfillcolor{eps2pgf_color}
\pgfpathmoveto{\pgfqpoint{0.117cm}{1.421cm}}
\pgfpathlineto{\pgfqpoint{0.682cm}{0.671cm}}
\pgfpathlineto{\pgfqpoint{1.246cm}{1.421cm}}
\pgfusepath{stroke}
\end{pgfscope}
\definecolor{eps2pgf_color}{gray}{0}\pgfsetstrokecolor{eps2pgf_color}\pgfsetfillcolor{eps2pgf_color}
\pgfpathmoveto{\pgfqpoint{0.273cm}{1.395cm}}
\pgfpathcurveto{\pgfqpoint{0.273cm}{1.432cm}}{\pgfqpoint{0.259cm}{1.467cm}}{\pgfqpoint{0.233cm}{1.492cm}}
\pgfpathcurveto{\pgfqpoint{0.207cm}{1.518cm}}{\pgfqpoint{0.173cm}{1.532cm}}{\pgfqpoint{0.137cm}{1.532cm}}
\pgfpathcurveto{\pgfqpoint{0.1cm}{1.532cm}}{\pgfqpoint{0.066cm}{1.518cm}}{\pgfqpoint{0.04cm}{1.492cm}}
\pgfpathcurveto{\pgfqpoint{0.014cm}{1.467cm}}{\pgfqpoint{0cm}{1.432cm}}{\pgfqpoint{0cm}{1.395cm}}
\pgfpathcurveto{\pgfqpoint{0cm}{1.359cm}}{\pgfqpoint{0.014cm}{1.324cm}}{\pgfqpoint{0.04cm}{1.299cm}}
\pgfpathcurveto{\pgfqpoint{0.066cm}{1.273cm}}{\pgfqpoint{0.1cm}{1.258cm}}{\pgfqpoint{0.137cm}{1.258cm}}
\pgfpathcurveto{\pgfqpoint{0.173cm}{1.258cm}}{\pgfqpoint{0.207cm}{1.273cm}}{\pgfqpoint{0.233cm}{1.299cm}}
\pgfpathcurveto{\pgfqpoint{0.259cm}{1.324cm}}{\pgfqpoint{0.273cm}{1.359cm}}{\pgfqpoint{0.273cm}{1.395cm}}
\pgfusepath{fill}
\begin{pgfscope}
\pgfsetdash{}{0cm}
\pgfsetlinewidth{0.818mm}
\pgfsetmiterlimit{7.0}
\pgfpathmoveto{\pgfqpoint{0.682cm}{0.671cm}}
\pgfpathlineto{\pgfqpoint{0.679cm}{1.418cm}}
\pgfusepath{stroke}
\end{pgfscope}
\pgfpathmoveto{\pgfqpoint{0.815cm}{1.399cm}}
\pgfpathcurveto{\pgfqpoint{0.815cm}{1.435cm}}{\pgfqpoint{0.801cm}{1.47cm}}{\pgfqpoint{0.775cm}{1.496cm}}
\pgfpathcurveto{\pgfqpoint{0.75cm}{1.521cm}}{\pgfqpoint{0.715cm}{1.536cm}}{\pgfqpoint{0.679cm}{1.536cm}}
\pgfpathcurveto{\pgfqpoint{0.643cm}{1.536cm}}{\pgfqpoint{0.608cm}{1.521cm}}{\pgfqpoint{0.582cm}{1.496cm}}
\pgfpathcurveto{\pgfqpoint{0.557cm}{1.47cm}}{\pgfqpoint{0.542cm}{1.435cm}}{\pgfqpoint{0.542cm}{1.399cm}}
\pgfpathcurveto{\pgfqpoint{0.542cm}{1.363cm}}{\pgfqpoint{0.557cm}{1.328cm}}{\pgfqpoint{0.582cm}{1.302cm}}
\pgfpathcurveto{\pgfqpoint{0.608cm}{1.276cm}}{\pgfqpoint{0.643cm}{1.262cm}}{\pgfqpoint{0.679cm}{1.262cm}}
\pgfpathcurveto{\pgfqpoint{0.715cm}{1.262cm}}{\pgfqpoint{0.75cm}{1.276cm}}{\pgfqpoint{0.775cm}{1.302cm}}
\pgfpathcurveto{\pgfqpoint{0.801cm}{1.328cm}}{\pgfqpoint{0.815cm}{1.363cm}}{\pgfqpoint{0.815cm}{1.399cm}}
\pgfusepath{fill}
\pgfpathmoveto{\pgfqpoint{1.345cm}{1.371cm}}
\pgfpathcurveto{\pgfqpoint{1.345cm}{1.408cm}}{\pgfqpoint{1.331cm}{1.442cm}}{\pgfqpoint{1.305cm}{1.468cm}}
\pgfpathcurveto{\pgfqpoint{1.28cm}{1.494cm}}{\pgfqpoint{1.245cm}{1.508cm}}{\pgfqpoint{1.209cm}{1.508cm}}
\pgfpathcurveto{\pgfqpoint{1.172cm}{1.508cm}}{\pgfqpoint{1.138cm}{1.494cm}}{\pgfqpoint{1.112cm}{1.468cm}}
\pgfpathcurveto{\pgfqpoint{1.087cm}{1.442cm}}{\pgfqpoint{1.072cm}{1.408cm}}{\pgfqpoint{1.072cm}{1.371cm}}
\pgfpathcurveto{\pgfqpoint{1.072cm}{1.335cm}}{\pgfqpoint{1.087cm}{1.3cm}}{\pgfqpoint{1.112cm}{1.274cm}}
\pgfpathcurveto{\pgfqpoint{1.138cm}{1.249cm}}{\pgfqpoint{1.172cm}{1.234cm}}{\pgfqpoint{1.209cm}{1.234cm}}
\pgfpathcurveto{\pgfqpoint{1.245cm}{1.234cm}}{\pgfqpoint{1.28cm}{1.249cm}}{\pgfqpoint{1.305cm}{1.274cm}}
\pgfpathcurveto{\pgfqpoint{1.331cm}{1.3cm}}{\pgfqpoint{1.345cm}{1.335cm}}{\pgfqpoint{1.345cm}{1.371cm}}
\pgfusepath{fill}
\begin{pgfscope}
\pgfsetdash{}{0cm}
\pgfsetlinewidth{0.818mm}
\pgfsetroundcap
\pgfsetmiterlimit{4.0}
\pgfpathmoveto{\pgfqpoint{0.682cm}{0.671cm}}
\pgfpathlineto{\pgfqpoint{0.682cm}{0.042cm}}
\pgfusepath{stroke}
\end{pgfscope}
\end{pgfscope}
\end{pgfscope}
\end{pgfscope}
\end{tikzpicture}}}_{M, \varepsilon} )$, $\varepsilon \in \mathcal{A}$, $M >
  0$, evaluated at an arbitrary time $t \geqslant 0$ is tight. Moreover, any
  limit measure $\mu$ satisfies for all $p \in [1, \infty)$
  \[ \mathbb{E}_{\mu} \| \varphi \|_{H^{- 1 / 2 - 2 \kappa} (\rho^2)}^{2 p}
     \lesssim 1 + \lambda^{3 p}, \qquad \mathbb{E}_{\mu} \| \zeta \|_{L^2
     (\rho^2)}^{2 p} \lesssim  \lambda^p + \lambda^{3p+4} + \lambda^{4p}, \]
  \[ \mathbb{E}_{\mu} \| \zeta \|_{H^{1 - 2 \kappa} (\rho^2)}^2 \lesssim
     \lambda^2 + \lambda^7, \qquad \mathbb{E}_{\mu} \| \zeta \|_{B^0_{4,
     \infty} (\rho)}^4 \lesssim \lambda + \lambda^6 . \]
\end{theorem}

\paragraph{Osterwalder--Schrader axioms.}The projection of a limit measure
$\mu$ onto the first component is the candidate $\Phi^4_3$ measure and we
denote it by $\nu$. Based on Theorem \ref{thm:main-int} we are able to show
that $\nu$ is translation invariant and reflection positive, see Section \ref{ss:OS1} and Section \ref{ss:OS2}. In addition, we prove that the measure is
non-Gaussian. To this end, we make use of the decomposition
$\varphi = X - \lambda X^{\!\resizebox{0.6em}{!}{
\begin{tikzpicture}
\pgfpathmoveto{\pgfqpoint{0cm}{-0.035cm}}
\pgfpathlineto{\pgfqpoint{1.376cm}{-0.035cm}}
\pgfpathlineto{\pgfqpoint{1.376cm}{1.552cm}}
\pgfpathlineto{\pgfqpoint{0cm}{1.552cm}}
\pgfpathclose
\pgfusepath{clip}
\begin{pgfscope}
\begin{pgfscope}
\pgfpathmoveto{\pgfqpoint{0cm}{-0.035cm}}
\pgfpathlineto{\pgfqpoint{1.376cm}{-0.035cm}}
\pgfpathlineto{\pgfqpoint{1.376cm}{1.552cm}}
\pgfpathlineto{\pgfqpoint{0cm}{1.552cm}}
\pgfpathclose
\pgfusepath{clip}
\begin{pgfscope}
\begin{pgfscope}
\pgfsetdash{}{0cm}
\pgfsetlinewidth{0.818mm}
\pgfsetroundcap
\pgfsetroundjoin
\pgfsetmiterlimit{7.0}
\definecolor{eps2pgf_color}{gray}{0}\pgfsetstrokecolor{eps2pgf_color}\pgfsetfillcolor{eps2pgf_color}
\pgfpathmoveto{\pgfqpoint{0.117cm}{1.421cm}}
\pgfpathlineto{\pgfqpoint{0.682cm}{0.671cm}}
\pgfpathlineto{\pgfqpoint{1.246cm}{1.421cm}}
\pgfusepath{stroke}
\end{pgfscope}
\definecolor{eps2pgf_color}{gray}{0}\pgfsetstrokecolor{eps2pgf_color}\pgfsetfillcolor{eps2pgf_color}
\pgfpathmoveto{\pgfqpoint{0.273cm}{1.395cm}}
\pgfpathcurveto{\pgfqpoint{0.273cm}{1.432cm}}{\pgfqpoint{0.259cm}{1.467cm}}{\pgfqpoint{0.233cm}{1.492cm}}
\pgfpathcurveto{\pgfqpoint{0.207cm}{1.518cm}}{\pgfqpoint{0.173cm}{1.532cm}}{\pgfqpoint{0.137cm}{1.532cm}}
\pgfpathcurveto{\pgfqpoint{0.1cm}{1.532cm}}{\pgfqpoint{0.066cm}{1.518cm}}{\pgfqpoint{0.04cm}{1.492cm}}
\pgfpathcurveto{\pgfqpoint{0.014cm}{1.467cm}}{\pgfqpoint{0cm}{1.432cm}}{\pgfqpoint{0cm}{1.395cm}}
\pgfpathcurveto{\pgfqpoint{0cm}{1.359cm}}{\pgfqpoint{0.014cm}{1.324cm}}{\pgfqpoint{0.04cm}{1.299cm}}
\pgfpathcurveto{\pgfqpoint{0.066cm}{1.273cm}}{\pgfqpoint{0.1cm}{1.258cm}}{\pgfqpoint{0.137cm}{1.258cm}}
\pgfpathcurveto{\pgfqpoint{0.173cm}{1.258cm}}{\pgfqpoint{0.207cm}{1.273cm}}{\pgfqpoint{0.233cm}{1.299cm}}
\pgfpathcurveto{\pgfqpoint{0.259cm}{1.324cm}}{\pgfqpoint{0.273cm}{1.359cm}}{\pgfqpoint{0.273cm}{1.395cm}}
\pgfusepath{fill}
\begin{pgfscope}
\pgfsetdash{}{0cm}
\pgfsetlinewidth{0.818mm}
\pgfsetmiterlimit{7.0}
\pgfpathmoveto{\pgfqpoint{0.682cm}{0.671cm}}
\pgfpathlineto{\pgfqpoint{0.679cm}{1.418cm}}
\pgfusepath{stroke}
\end{pgfscope}
\pgfpathmoveto{\pgfqpoint{0.815cm}{1.399cm}}
\pgfpathcurveto{\pgfqpoint{0.815cm}{1.435cm}}{\pgfqpoint{0.801cm}{1.47cm}}{\pgfqpoint{0.775cm}{1.496cm}}
\pgfpathcurveto{\pgfqpoint{0.75cm}{1.521cm}}{\pgfqpoint{0.715cm}{1.536cm}}{\pgfqpoint{0.679cm}{1.536cm}}
\pgfpathcurveto{\pgfqpoint{0.643cm}{1.536cm}}{\pgfqpoint{0.608cm}{1.521cm}}{\pgfqpoint{0.582cm}{1.496cm}}
\pgfpathcurveto{\pgfqpoint{0.557cm}{1.47cm}}{\pgfqpoint{0.542cm}{1.435cm}}{\pgfqpoint{0.542cm}{1.399cm}}
\pgfpathcurveto{\pgfqpoint{0.542cm}{1.363cm}}{\pgfqpoint{0.557cm}{1.328cm}}{\pgfqpoint{0.582cm}{1.302cm}}
\pgfpathcurveto{\pgfqpoint{0.608cm}{1.276cm}}{\pgfqpoint{0.643cm}{1.262cm}}{\pgfqpoint{0.679cm}{1.262cm}}
\pgfpathcurveto{\pgfqpoint{0.715cm}{1.262cm}}{\pgfqpoint{0.75cm}{1.276cm}}{\pgfqpoint{0.775cm}{1.302cm}}
\pgfpathcurveto{\pgfqpoint{0.801cm}{1.328cm}}{\pgfqpoint{0.815cm}{1.363cm}}{\pgfqpoint{0.815cm}{1.399cm}}
\pgfusepath{fill}
\pgfpathmoveto{\pgfqpoint{1.345cm}{1.371cm}}
\pgfpathcurveto{\pgfqpoint{1.345cm}{1.408cm}}{\pgfqpoint{1.331cm}{1.442cm}}{\pgfqpoint{1.305cm}{1.468cm}}
\pgfpathcurveto{\pgfqpoint{1.28cm}{1.494cm}}{\pgfqpoint{1.245cm}{1.508cm}}{\pgfqpoint{1.209cm}{1.508cm}}
\pgfpathcurveto{\pgfqpoint{1.172cm}{1.508cm}}{\pgfqpoint{1.138cm}{1.494cm}}{\pgfqpoint{1.112cm}{1.468cm}}
\pgfpathcurveto{\pgfqpoint{1.087cm}{1.442cm}}{\pgfqpoint{1.072cm}{1.408cm}}{\pgfqpoint{1.072cm}{1.371cm}}
\pgfpathcurveto{\pgfqpoint{1.072cm}{1.335cm}}{\pgfqpoint{1.087cm}{1.3cm}}{\pgfqpoint{1.112cm}{1.274cm}}
\pgfpathcurveto{\pgfqpoint{1.138cm}{1.249cm}}{\pgfqpoint{1.172cm}{1.234cm}}{\pgfqpoint{1.209cm}{1.234cm}}
\pgfpathcurveto{\pgfqpoint{1.245cm}{1.234cm}}{\pgfqpoint{1.28cm}{1.249cm}}{\pgfqpoint{1.305cm}{1.274cm}}
\pgfpathcurveto{\pgfqpoint{1.331cm}{1.3cm}}{\pgfqpoint{1.345cm}{1.335cm}}{\pgfqpoint{1.345cm}{1.371cm}}
\pgfusepath{fill}
\begin{pgfscope}
\pgfsetdash{}{0cm}
\pgfsetlinewidth{0.818mm}
\pgfsetroundcap
\pgfsetmiterlimit{4.0}
\pgfpathmoveto{\pgfqpoint{0.682cm}{0.671cm}}
\pgfpathlineto{\pgfqpoint{0.682cm}{0.042cm}}
\pgfusepath{stroke}
\end{pgfscope}
\end{pgfscope}
\end{pgfscope}
\end{pgfscope}
\end{tikzpicture}}} + \zeta$ together with the moment bounds
from Theorem \ref{thm:main-int}. Since $X$ is Gaussian whereas $X^{\!\resizebox{0.6em}{!}{
\begin{tikzpicture}
\pgfpathmoveto{\pgfqpoint{0cm}{-0.035cm}}
\pgfpathlineto{\pgfqpoint{1.376cm}{-0.035cm}}
\pgfpathlineto{\pgfqpoint{1.376cm}{1.552cm}}
\pgfpathlineto{\pgfqpoint{0cm}{1.552cm}}
\pgfpathclose
\pgfusepath{clip}
\begin{pgfscope}
\begin{pgfscope}
\pgfpathmoveto{\pgfqpoint{0cm}{-0.035cm}}
\pgfpathlineto{\pgfqpoint{1.376cm}{-0.035cm}}
\pgfpathlineto{\pgfqpoint{1.376cm}{1.552cm}}
\pgfpathlineto{\pgfqpoint{0cm}{1.552cm}}
\pgfpathclose
\pgfusepath{clip}
\begin{pgfscope}
\begin{pgfscope}
\pgfsetdash{}{0cm}
\pgfsetlinewidth{0.818mm}
\pgfsetroundcap
\pgfsetroundjoin
\pgfsetmiterlimit{7.0}
\definecolor{eps2pgf_color}{gray}{0}\pgfsetstrokecolor{eps2pgf_color}\pgfsetfillcolor{eps2pgf_color}
\pgfpathmoveto{\pgfqpoint{0.117cm}{1.421cm}}
\pgfpathlineto{\pgfqpoint{0.682cm}{0.671cm}}
\pgfpathlineto{\pgfqpoint{1.246cm}{1.421cm}}
\pgfusepath{stroke}
\end{pgfscope}
\definecolor{eps2pgf_color}{gray}{0}\pgfsetstrokecolor{eps2pgf_color}\pgfsetfillcolor{eps2pgf_color}
\pgfpathmoveto{\pgfqpoint{0.273cm}{1.395cm}}
\pgfpathcurveto{\pgfqpoint{0.273cm}{1.432cm}}{\pgfqpoint{0.259cm}{1.467cm}}{\pgfqpoint{0.233cm}{1.492cm}}
\pgfpathcurveto{\pgfqpoint{0.207cm}{1.518cm}}{\pgfqpoint{0.173cm}{1.532cm}}{\pgfqpoint{0.137cm}{1.532cm}}
\pgfpathcurveto{\pgfqpoint{0.1cm}{1.532cm}}{\pgfqpoint{0.066cm}{1.518cm}}{\pgfqpoint{0.04cm}{1.492cm}}
\pgfpathcurveto{\pgfqpoint{0.014cm}{1.467cm}}{\pgfqpoint{0cm}{1.432cm}}{\pgfqpoint{0cm}{1.395cm}}
\pgfpathcurveto{\pgfqpoint{0cm}{1.359cm}}{\pgfqpoint{0.014cm}{1.324cm}}{\pgfqpoint{0.04cm}{1.299cm}}
\pgfpathcurveto{\pgfqpoint{0.066cm}{1.273cm}}{\pgfqpoint{0.1cm}{1.258cm}}{\pgfqpoint{0.137cm}{1.258cm}}
\pgfpathcurveto{\pgfqpoint{0.173cm}{1.258cm}}{\pgfqpoint{0.207cm}{1.273cm}}{\pgfqpoint{0.233cm}{1.299cm}}
\pgfpathcurveto{\pgfqpoint{0.259cm}{1.324cm}}{\pgfqpoint{0.273cm}{1.359cm}}{\pgfqpoint{0.273cm}{1.395cm}}
\pgfusepath{fill}
\begin{pgfscope}
\pgfsetdash{}{0cm}
\pgfsetlinewidth{0.818mm}
\pgfsetmiterlimit{7.0}
\pgfpathmoveto{\pgfqpoint{0.682cm}{0.671cm}}
\pgfpathlineto{\pgfqpoint{0.679cm}{1.418cm}}
\pgfusepath{stroke}
\end{pgfscope}
\pgfpathmoveto{\pgfqpoint{0.815cm}{1.399cm}}
\pgfpathcurveto{\pgfqpoint{0.815cm}{1.435cm}}{\pgfqpoint{0.801cm}{1.47cm}}{\pgfqpoint{0.775cm}{1.496cm}}
\pgfpathcurveto{\pgfqpoint{0.75cm}{1.521cm}}{\pgfqpoint{0.715cm}{1.536cm}}{\pgfqpoint{0.679cm}{1.536cm}}
\pgfpathcurveto{\pgfqpoint{0.643cm}{1.536cm}}{\pgfqpoint{0.608cm}{1.521cm}}{\pgfqpoint{0.582cm}{1.496cm}}
\pgfpathcurveto{\pgfqpoint{0.557cm}{1.47cm}}{\pgfqpoint{0.542cm}{1.435cm}}{\pgfqpoint{0.542cm}{1.399cm}}
\pgfpathcurveto{\pgfqpoint{0.542cm}{1.363cm}}{\pgfqpoint{0.557cm}{1.328cm}}{\pgfqpoint{0.582cm}{1.302cm}}
\pgfpathcurveto{\pgfqpoint{0.608cm}{1.276cm}}{\pgfqpoint{0.643cm}{1.262cm}}{\pgfqpoint{0.679cm}{1.262cm}}
\pgfpathcurveto{\pgfqpoint{0.715cm}{1.262cm}}{\pgfqpoint{0.75cm}{1.276cm}}{\pgfqpoint{0.775cm}{1.302cm}}
\pgfpathcurveto{\pgfqpoint{0.801cm}{1.328cm}}{\pgfqpoint{0.815cm}{1.363cm}}{\pgfqpoint{0.815cm}{1.399cm}}
\pgfusepath{fill}
\pgfpathmoveto{\pgfqpoint{1.345cm}{1.371cm}}
\pgfpathcurveto{\pgfqpoint{1.345cm}{1.408cm}}{\pgfqpoint{1.331cm}{1.442cm}}{\pgfqpoint{1.305cm}{1.468cm}}
\pgfpathcurveto{\pgfqpoint{1.28cm}{1.494cm}}{\pgfqpoint{1.245cm}{1.508cm}}{\pgfqpoint{1.209cm}{1.508cm}}
\pgfpathcurveto{\pgfqpoint{1.172cm}{1.508cm}}{\pgfqpoint{1.138cm}{1.494cm}}{\pgfqpoint{1.112cm}{1.468cm}}
\pgfpathcurveto{\pgfqpoint{1.087cm}{1.442cm}}{\pgfqpoint{1.072cm}{1.408cm}}{\pgfqpoint{1.072cm}{1.371cm}}
\pgfpathcurveto{\pgfqpoint{1.072cm}{1.335cm}}{\pgfqpoint{1.087cm}{1.3cm}}{\pgfqpoint{1.112cm}{1.274cm}}
\pgfpathcurveto{\pgfqpoint{1.138cm}{1.249cm}}{\pgfqpoint{1.172cm}{1.234cm}}{\pgfqpoint{1.209cm}{1.234cm}}
\pgfpathcurveto{\pgfqpoint{1.245cm}{1.234cm}}{\pgfqpoint{1.28cm}{1.249cm}}{\pgfqpoint{1.305cm}{1.274cm}}
\pgfpathcurveto{\pgfqpoint{1.331cm}{1.3cm}}{\pgfqpoint{1.345cm}{1.335cm}}{\pgfqpoint{1.345cm}{1.371cm}}
\pgfusepath{fill}
\begin{pgfscope}
\pgfsetdash{}{0cm}
\pgfsetlinewidth{0.818mm}
\pgfsetroundcap
\pgfsetmiterlimit{4.0}
\pgfpathmoveto{\pgfqpoint{0.682cm}{0.671cm}}
\pgfpathlineto{\pgfqpoint{0.682cm}{0.042cm}}
\pgfusepath{stroke}
\end{pgfscope}
\end{pgfscope}
\end{pgfscope}
\end{pgfscope}
\end{tikzpicture}}}$
is not, the idea is to use the regularity of $\zeta$ to conclude that it
cannot compensate $X^{\!\resizebox{0.6em}{!}{
\begin{tikzpicture}
\pgfpathmoveto{\pgfqpoint{0cm}{-0.035cm}}
\pgfpathlineto{\pgfqpoint{1.376cm}{-0.035cm}}
\pgfpathlineto{\pgfqpoint{1.376cm}{1.552cm}}
\pgfpathlineto{\pgfqpoint{0cm}{1.552cm}}
\pgfpathclose
\pgfusepath{clip}
\begin{pgfscope}
\begin{pgfscope}
\pgfpathmoveto{\pgfqpoint{0cm}{-0.035cm}}
\pgfpathlineto{\pgfqpoint{1.376cm}{-0.035cm}}
\pgfpathlineto{\pgfqpoint{1.376cm}{1.552cm}}
\pgfpathlineto{\pgfqpoint{0cm}{1.552cm}}
\pgfpathclose
\pgfusepath{clip}
\begin{pgfscope}
\begin{pgfscope}
\pgfsetdash{}{0cm}
\pgfsetlinewidth{0.818mm}
\pgfsetroundcap
\pgfsetroundjoin
\pgfsetmiterlimit{7.0}
\definecolor{eps2pgf_color}{gray}{0}\pgfsetstrokecolor{eps2pgf_color}\pgfsetfillcolor{eps2pgf_color}
\pgfpathmoveto{\pgfqpoint{0.117cm}{1.421cm}}
\pgfpathlineto{\pgfqpoint{0.682cm}{0.671cm}}
\pgfpathlineto{\pgfqpoint{1.246cm}{1.421cm}}
\pgfusepath{stroke}
\end{pgfscope}
\definecolor{eps2pgf_color}{gray}{0}\pgfsetstrokecolor{eps2pgf_color}\pgfsetfillcolor{eps2pgf_color}
\pgfpathmoveto{\pgfqpoint{0.273cm}{1.395cm}}
\pgfpathcurveto{\pgfqpoint{0.273cm}{1.432cm}}{\pgfqpoint{0.259cm}{1.467cm}}{\pgfqpoint{0.233cm}{1.492cm}}
\pgfpathcurveto{\pgfqpoint{0.207cm}{1.518cm}}{\pgfqpoint{0.173cm}{1.532cm}}{\pgfqpoint{0.137cm}{1.532cm}}
\pgfpathcurveto{\pgfqpoint{0.1cm}{1.532cm}}{\pgfqpoint{0.066cm}{1.518cm}}{\pgfqpoint{0.04cm}{1.492cm}}
\pgfpathcurveto{\pgfqpoint{0.014cm}{1.467cm}}{\pgfqpoint{0cm}{1.432cm}}{\pgfqpoint{0cm}{1.395cm}}
\pgfpathcurveto{\pgfqpoint{0cm}{1.359cm}}{\pgfqpoint{0.014cm}{1.324cm}}{\pgfqpoint{0.04cm}{1.299cm}}
\pgfpathcurveto{\pgfqpoint{0.066cm}{1.273cm}}{\pgfqpoint{0.1cm}{1.258cm}}{\pgfqpoint{0.137cm}{1.258cm}}
\pgfpathcurveto{\pgfqpoint{0.173cm}{1.258cm}}{\pgfqpoint{0.207cm}{1.273cm}}{\pgfqpoint{0.233cm}{1.299cm}}
\pgfpathcurveto{\pgfqpoint{0.259cm}{1.324cm}}{\pgfqpoint{0.273cm}{1.359cm}}{\pgfqpoint{0.273cm}{1.395cm}}
\pgfusepath{fill}
\begin{pgfscope}
\pgfsetdash{}{0cm}
\pgfsetlinewidth{0.818mm}
\pgfsetmiterlimit{7.0}
\pgfpathmoveto{\pgfqpoint{0.682cm}{0.671cm}}
\pgfpathlineto{\pgfqpoint{0.679cm}{1.418cm}}
\pgfusepath{stroke}
\end{pgfscope}
\pgfpathmoveto{\pgfqpoint{0.815cm}{1.399cm}}
\pgfpathcurveto{\pgfqpoint{0.815cm}{1.435cm}}{\pgfqpoint{0.801cm}{1.47cm}}{\pgfqpoint{0.775cm}{1.496cm}}
\pgfpathcurveto{\pgfqpoint{0.75cm}{1.521cm}}{\pgfqpoint{0.715cm}{1.536cm}}{\pgfqpoint{0.679cm}{1.536cm}}
\pgfpathcurveto{\pgfqpoint{0.643cm}{1.536cm}}{\pgfqpoint{0.608cm}{1.521cm}}{\pgfqpoint{0.582cm}{1.496cm}}
\pgfpathcurveto{\pgfqpoint{0.557cm}{1.47cm}}{\pgfqpoint{0.542cm}{1.435cm}}{\pgfqpoint{0.542cm}{1.399cm}}
\pgfpathcurveto{\pgfqpoint{0.542cm}{1.363cm}}{\pgfqpoint{0.557cm}{1.328cm}}{\pgfqpoint{0.582cm}{1.302cm}}
\pgfpathcurveto{\pgfqpoint{0.608cm}{1.276cm}}{\pgfqpoint{0.643cm}{1.262cm}}{\pgfqpoint{0.679cm}{1.262cm}}
\pgfpathcurveto{\pgfqpoint{0.715cm}{1.262cm}}{\pgfqpoint{0.75cm}{1.276cm}}{\pgfqpoint{0.775cm}{1.302cm}}
\pgfpathcurveto{\pgfqpoint{0.801cm}{1.328cm}}{\pgfqpoint{0.815cm}{1.363cm}}{\pgfqpoint{0.815cm}{1.399cm}}
\pgfusepath{fill}
\pgfpathmoveto{\pgfqpoint{1.345cm}{1.371cm}}
\pgfpathcurveto{\pgfqpoint{1.345cm}{1.408cm}}{\pgfqpoint{1.331cm}{1.442cm}}{\pgfqpoint{1.305cm}{1.468cm}}
\pgfpathcurveto{\pgfqpoint{1.28cm}{1.494cm}}{\pgfqpoint{1.245cm}{1.508cm}}{\pgfqpoint{1.209cm}{1.508cm}}
\pgfpathcurveto{\pgfqpoint{1.172cm}{1.508cm}}{\pgfqpoint{1.138cm}{1.494cm}}{\pgfqpoint{1.112cm}{1.468cm}}
\pgfpathcurveto{\pgfqpoint{1.087cm}{1.442cm}}{\pgfqpoint{1.072cm}{1.408cm}}{\pgfqpoint{1.072cm}{1.371cm}}
\pgfpathcurveto{\pgfqpoint{1.072cm}{1.335cm}}{\pgfqpoint{1.087cm}{1.3cm}}{\pgfqpoint{1.112cm}{1.274cm}}
\pgfpathcurveto{\pgfqpoint{1.138cm}{1.249cm}}{\pgfqpoint{1.172cm}{1.234cm}}{\pgfqpoint{1.209cm}{1.234cm}}
\pgfpathcurveto{\pgfqpoint{1.245cm}{1.234cm}}{\pgfqpoint{1.28cm}{1.249cm}}{\pgfqpoint{1.305cm}{1.274cm}}
\pgfpathcurveto{\pgfqpoint{1.331cm}{1.3cm}}{\pgfqpoint{1.345cm}{1.335cm}}{\pgfqpoint{1.345cm}{1.371cm}}
\pgfusepath{fill}
\begin{pgfscope}
\pgfsetdash{}{0cm}
\pgfsetlinewidth{0.818mm}
\pgfsetroundcap
\pgfsetmiterlimit{4.0}
\pgfpathmoveto{\pgfqpoint{0.682cm}{0.671cm}}
\pgfpathlineto{\pgfqpoint{0.682cm}{0.042cm}}
\pgfusepath{stroke}
\end{pgfscope}
\end{pgfscope}
\end{pgfscope}
\end{pgfscope}
\end{tikzpicture}}}$ which is less regular. In particular, we
show that the connected $4$-point function is nonzero, see Section \ref{ss:nonG}.

It remains to discuss a stretched exponential integrability of $\varphi$, leading
to the distribution property  shown in Section \ref{ss:OS0}. More precisely, we show the following result
which can be found in Proposition~\ref{lemma:int-bound}.

\begin{proposition}\label{prop:exp}
  Let $\rho$ be a weight such that $\rho^{\iota} \in L^{4, 0}$ for some $\iota
  \in (0, 1)$. For every $\kappa \in (0, 1)$ small there exists $\upsilon = O
  (\kappa) > 0$ small such that
  \[ \int_{\mathcal{S}'(\mathbb{R}^{3})} \exp\{{\beta \| \varphi \|_{H^{- 1 / 2 - 2
     \kappa} (\rho^2)}^{1 - \upsilon}} \} \nu (\mathrm{d}\varphi)< \infty \]
  provided $\beta > 0$ is chosen sufficiently small.
\end{proposition}

In order to obtain this bound we revisit the bounds from Theorem
\ref{th:energy-estimate-int} and track the precise dependence of the
polynomial $Q_{\rho} (\mathbb{X}_{M, \varepsilon})$ on the right hand side of the estimate on
the quantity $\| \mathbb{X}_{M, \varepsilon} \|$ which will be defined  through \eqref{eq:XX1}, \eqref{eq:XX2}, \eqref{eq:XX3} below taking into account the number of copies of $X$ appearing in each stochastic object. However, the estimates in
Theorem \ref{th:energy-estimate-int} are not optimal and consequently the
power of $\| \mathbb{X}_{M, \varepsilon} \|$ in Theorem
\ref{th:energy-estimate-int} is too
large. To optimize we introduce a large momentum cut-off $\llbracket X^3_{M,
\varepsilon} \rrbracket_{\leqslant}$ given by a parameter $K > 0$ and let
$\llbracket X^3_{M, \varepsilon} \rrbracket_{>} \assign \llbracket X^3_{M,
\varepsilon} \rrbracket - \llbracket X^3_{M, \varepsilon}
\rrbracket_{\leqslant}$. Then we modify the dynamics of $Y_{M, \varepsilon}$
to
\[ \LL_{\varepsilon} Y_{M, \varepsilon} = - \llbracket
   X_{M, \varepsilon}^3 \rrbracket_{>} - 3 \lambda
   (\mathscr{U}^{\varepsilon}_{>} \llbracket X_{M, \varepsilon}^2 \rrbracket)
   \succ Y_{M, \varepsilon}, \]
which allows for refined bounds on $Y_{M, \varepsilon}$, yielding optimal
powers of $\| \mathbb{X}_{M, \varepsilon} \|$.

\paragraph{Integration by parts formula.}
The uniform energy estimates from Theorem~\ref{thm:main-int} and Proposition~\ref{prop:exp} are enough
to obtain tightness of the approximate measures and to show that any
accumulation point satisfies the distribution property,  translation invariance, reflection positivity and
non-Gaussianity. However, they do not provide sufficient regularity in order to
identify the continuum dynamics or to establish the hierarchy of
Dyson--Schwinger equations providing relations of various $n$-point
correlation functions. This can be seen easily since neither the resonant
product $\llbracket X_{M, \varepsilon}^2 \rrbracket \circ \phi_{M,
\varepsilon}$ nor $\llbracket X_{M, \varepsilon}^2 \rrbracket \circ \psi_{M,
\varepsilon}$ is well-defined in the limit.
Another and even more severe difficulty lies in the fact that the  third Wick power $\llbracket X^3 \rrbracket$ only exists as a space-time distribution and is  not a well-defined random variable under the $\Phi^{4}_{3}$ measure, cf.~\cite{albeverio_remark_2006}.

To overcome the first issue, we  introduce a new
paracontrolled ansatz
$  \chi_{M, \varepsilon} \assign \phi_{M, \varepsilon} + 3\lambda X_{M,
  \varepsilon}^{\!\resizebox{0.6em}{!}{
\begin{tikzpicture}
\pgfpathmoveto{\pgfqpoint{0cm}{0cm}}
\pgfpathlineto{\pgfqpoint{1.376cm}{0cm}}
\pgfpathlineto{\pgfqpoint{1.376cm}{1.588cm}}
\pgfpathlineto{\pgfqpoint{0cm}{1.588cm}}
\pgfpathclose
\pgfusepath{clip}
\begin{pgfscope}
\begin{pgfscope}
\pgfpathmoveto{\pgfqpoint{0cm}{0cm}}
\pgfpathlineto{\pgfqpoint{1.376cm}{0cm}}
\pgfpathlineto{\pgfqpoint{1.376cm}{1.588cm}}
\pgfpathlineto{\pgfqpoint{0cm}{1.588cm}}
\pgfpathclose
\pgfusepath{clip}
\begin{pgfscope}
\begin{pgfscope}
\definecolor{eps2pgf_color}{gray}{0.976471}\pgfsetstrokecolor{eps2pgf_color}\pgfsetfillcolor{eps2pgf_color}
\pgfpathmoveto{\pgfqpoint{0cm}{0cm}}
\pgfpathlineto{\pgfqpoint{1.376cm}{0cm}}
\pgfpathlineto{\pgfqpoint{1.376cm}{1.588cm}}
\pgfpathlineto{\pgfqpoint{0cm}{1.588cm}}
\pgfpathclose
\pgfusepath{fill}
\end{pgfscope}
\begin{pgfscope}
\pgfsetdash{}{0cm}
\pgfsetlinewidth{0.818mm}
\pgfsetroundcap
\pgfsetroundjoin
\pgfsetmiterlimit{7.0}
\definecolor{eps2pgf_color}{gray}{0}\pgfsetstrokecolor{eps2pgf_color}\pgfsetfillcolor{eps2pgf_color}
\pgfpathmoveto{\pgfqpoint{0.117cm}{1.476cm}}
\pgfpathlineto{\pgfqpoint{0.682cm}{0.726cm}}
\pgfpathlineto{\pgfqpoint{1.246cm}{1.476cm}}
\pgfusepath{stroke}
\end{pgfscope}
\definecolor{eps2pgf_color}{gray}{0}\pgfsetstrokecolor{eps2pgf_color}\pgfsetfillcolor{eps2pgf_color}
\pgfpathmoveto{\pgfqpoint{0.273cm}{1.451cm}}
\pgfpathcurveto{\pgfqpoint{0.273cm}{1.487cm}}{\pgfqpoint{0.259cm}{1.522cm}}{\pgfqpoint{0.233cm}{1.547cm}}
\pgfpathcurveto{\pgfqpoint{0.207cm}{1.573cm}}{\pgfqpoint{0.173cm}{1.588cm}}{\pgfqpoint{0.137cm}{1.588cm}}
\pgfpathcurveto{\pgfqpoint{0.1cm}{1.588cm}}{\pgfqpoint{0.066cm}{1.573cm}}{\pgfqpoint{0.04cm}{1.547cm}}
\pgfpathcurveto{\pgfqpoint{0.014cm}{1.522cm}}{\pgfqpoint{0cm}{1.487cm}}{\pgfqpoint{0cm}{1.451cm}}
\pgfpathcurveto{\pgfqpoint{0cm}{1.414cm}}{\pgfqpoint{0.014cm}{1.379cm}}{\pgfqpoint{0.04cm}{1.354cm}}
\pgfpathcurveto{\pgfqpoint{0.066cm}{1.328cm}}{\pgfqpoint{0.1cm}{1.314cm}}{\pgfqpoint{0.137cm}{1.314cm}}
\pgfpathcurveto{\pgfqpoint{0.173cm}{1.314cm}}{\pgfqpoint{0.207cm}{1.328cm}}{\pgfqpoint{0.233cm}{1.354cm}}
\pgfpathcurveto{\pgfqpoint{0.259cm}{1.379cm}}{\pgfqpoint{0.273cm}{1.414cm}}{\pgfqpoint{0.273cm}{1.451cm}}
\pgfusepath{fill}
\pgfpathmoveto{\pgfqpoint{1.345cm}{1.426cm}}
\pgfpathcurveto{\pgfqpoint{1.345cm}{1.463cm}}{\pgfqpoint{1.331cm}{1.497cm}}{\pgfqpoint{1.305cm}{1.523cm}}
\pgfpathcurveto{\pgfqpoint{1.28cm}{1.549cm}}{\pgfqpoint{1.245cm}{1.563cm}}{\pgfqpoint{1.209cm}{1.563cm}}
\pgfpathcurveto{\pgfqpoint{1.172cm}{1.563cm}}{\pgfqpoint{1.138cm}{1.549cm}}{\pgfqpoint{1.112cm}{1.523cm}}
\pgfpathcurveto{\pgfqpoint{1.087cm}{1.497cm}}{\pgfqpoint{1.072cm}{1.463cm}}{\pgfqpoint{1.072cm}{1.426cm}}
\pgfpathcurveto{\pgfqpoint{1.072cm}{1.39cm}}{\pgfqpoint{1.087cm}{1.355cm}}{\pgfqpoint{1.112cm}{1.329cm}}
\pgfpathcurveto{\pgfqpoint{1.138cm}{1.304cm}}{\pgfqpoint{1.172cm}{1.289cm}}{\pgfqpoint{1.209cm}{1.289cm}}
\pgfpathcurveto{\pgfqpoint{1.245cm}{1.289cm}}{\pgfqpoint{1.28cm}{1.304cm}}{\pgfqpoint{1.305cm}{1.329cm}}
\pgfpathcurveto{\pgfqpoint{1.331cm}{1.355cm}}{\pgfqpoint{1.345cm}{1.39cm}}{\pgfqpoint{1.345cm}{1.426cm}}
\pgfusepath{fill}
\begin{pgfscope}
\pgfsetdash{}{0cm}
\pgfsetlinewidth{0.818mm}
\pgfsetroundcap
\pgfsetmiterlimit{4.0}
\pgfpathmoveto{\pgfqpoint{0.682cm}{0.726cm}}
\pgfpathlineto{\pgfqpoint{0.682cm}{0.097cm}}
\pgfusepath{stroke}
\end{pgfscope}
\end{pgfscope}
\end{pgfscope}
\end{pgfscope}
\end{tikzpicture}}} \succ \phi_{M, \varepsilon}$
and show that $\chi_{M,\varepsilon}$ possesses enough regularity uniformly in $M,\varepsilon$ in order to pass to the limit in the resonant product $\llbracket X^{2}_{M,\varepsilon}\rrbracket\circ \chi_{M,\varepsilon}$.
Namely,  we establish uniform bounds for $\chi_{M,\varepsilon}$ in $L^1_T B_{1, 1}^{1 + 3 \kappa,
     \varepsilon}(\rho^{4})$. This not only allows to give meaning to the critical resonant product in the continuum, but  it  also leads to a uniform time regularity of the processes $\varphi_{M,\varepsilon}$.
We obtain the following result proved below as Theorem~\ref{thm:phitight}.

\begin{theorem}
  \label{thm:phitight-int}Let $\beta \in (0, 1 / 4)$ and $\sigma\in (0,1)$. Then  for
  all $p \in [1, \infty)$ and $\tau \in (0, T)$
  \[ \sup_{\varepsilon \in \mathcal{A}, M > 0} \mathbb{E} \| \varphi_{M,
     \varepsilon} \|^{2 p}_{W^{\beta, 1}_T B_{1, 1}^{- 1 - 3 \kappa,\varepsilon} (\rho^{4
     + \sigma})} + \sup_{\varepsilon \in \mathcal{A}, M > 0} \mathbb{E} \|
     \varphi_{M, \varepsilon} \|^{2 p}_{L^{\infty}_{\tau, T} H^{- 1 / 2 -2
     \kappa,\varepsilon} (\rho^2)}  < \infty, \]
  where $L^{\infty}_{\tau, T} H^{- 1 / 2 -2 \kappa,\varepsilon} (\rho^2) = L^{\infty}
  (\tau, T ; H^{- 1 / 2 -2 \kappa,\varepsilon} (\rho^2))$.
\end{theorem}

This additional time regularity  is then used in order to treat the second issue raised above
and  to construct a renormalized cubic term $\llbracket
\varphi^3 \rrbracket$. More precisely, we derive an explicit formula for $\llbracket \varphi^{3}\rrbracket$ including  $\llbracket X^3 \rrbracket$ as a space-time distribution, where  {\em time} indeed means the fictitious {\em stochastic} time variable introduced by the stochastic quantization, nonexistent under the $\Phi^{4}_{3}$ measure. In order to control $\llbracket X^3 \rrbracket$ we re-introduce the stochastic time and use stationarity together with the above mentioned time regularity.
Finally, we derive an integration by parts formula
 leading to the hierarchy of Dyson--Schwinger equations connecting the correlation functions. 
To this end, we recall that a cylinder function $F$ on $\mathcal{S}' (\mathbb{R}^3)$ has the form $F
(\varphi) = \Phi (\varphi (f_1), \ldots, \varphi (f_n))$ where $\Phi :
\mathbb{R}^n \rightarrow \mathbb{R}$ and $f_1, \ldots, f_n \in \mathcal{S}
(\mathbb{R}^3)$. Loosely stated, the result proved in Theorem \ref{thm:ibp} says the following.

\begin{theorem}
  \label{thm:ibp-int}
  Let $F : \mathcal{S}' (\mathbb{R}^3) \rightarrow \mathbb{R}$
  be a cylinder function such that
  \[ | F (\varphi) | + \| \mathD F (\varphi) \|_{B_{\infty, \infty}^{1 + 3
     \kappa} (\rho^{- 4 - \sigma})} \leqslant C_F \| \varphi \|_{H^{- 1 / 2 -
    2 \kappa} (\rho^2)}^n \]
  for some $n \in \mathbb{N}$, where $\mathD F (\varphi)$  the $L^2$-gradient of $F$. Any accumulation point $\nu$ of the sequence $(\nu_{M,
  \varepsilon} \circ (\mathcal{E}^{\varepsilon})^{- 1})_{M, \varepsilon}$
  satisfies for all $f\in \mathcal{S}(\mathbb{R}^{3})$
  \begin{equation*}
    \int \langle\mathD F (\varphi),f\rangle \nu (\mathd \varphi) = 2 \int \langle(m^2 - \Delta)
    \varphi,f\rangle F (\varphi) \nu (\mathd \varphi) + 2\lambda \langle\mathcal{J}_{\nu} (F),f\rangle,
  \end{equation*}
  where for a smooth $h : \mathbb{R} \rightarrow \mathbb{R}$ 
with $\tmop{supp} h \subset [\tau, T]$ for some $0 < \tau < T < \infty$ and
 $\int_{\mathbb{R}} h (t) \mathd t = 1$ we have for all $f\in \mathcal{S}(\mathbb{R}^{3})$
  \[ \langle\mathcal{J}_{\nu} (F),f\rangle
   =\mathbb{E}_{\nu} \left[ \int_{\mathbb{R}} h (t) F (\varphi (t))\langle \llbracket
   \varphi^3 \rrbracket (t) ,f\rangle\mathd t \right]
    \]
    and $\llbracket \varphi^{3}\rrbracket$ is given by an explicit formula, namely, \eqref{eq:phi3}.
\end{theorem}

In addition, we are able to characterize $\mathcal{J}_{\nu}(F)$ in the spirit of the operator product expansion, see Lemma \ref{lem:OPE}.

\section{Construction of the Euclidean $\Phi^4$ field theory}

\label{sec:tight}This section is devoted to our main result. More precisely,
we consider \eqref{eq:moll} which is a discrete approximation of {\eqref{eq:P4}} posed on a periodic
lattice $\Lambda_{M, \varepsilon}$. For every $\varepsilon \in (0, 1)$
and $M > 0$ \eqref{eq:moll} possesses a unique invariant measure that is the Gibbs measure $\nu_{M,\varepsilon}$ given by \eqref{eq:gibbs}.
We derive new estimates on stationary solutions sampled from these measures
which hold true uniformly in $\varepsilon$ and $M$. As a consequence, we
obtain tightness of the invariant measures while sending both the mesh size as
well as the volume to their respective limits, i.e. $\varepsilon \rightarrow
0$, $M \rightarrow \infty$.

\subsection{Stochastic terms}

\label{ssec:stoch}
Recall that the stochastic objects $X_{M,\varepsilon},\llbracket X_{M,\varepsilon}^{2}\rrbracket, \llbracket X_{M,\varepsilon}^{3}\rrbracket$ and $X_{M, \varepsilon}^{\!\resizebox{0.6em}{!}{
\begin{tikzpicture}
\pgfpathmoveto{\pgfqpoint{0cm}{-0.035cm}}
\pgfpathlineto{\pgfqpoint{1.376cm}{-0.035cm}}
\pgfpathlineto{\pgfqpoint{1.376cm}{1.552cm}}
\pgfpathlineto{\pgfqpoint{0cm}{1.552cm}}
\pgfpathclose
\pgfusepath{clip}
\begin{pgfscope}
\begin{pgfscope}
\pgfpathmoveto{\pgfqpoint{0cm}{-0.035cm}}
\pgfpathlineto{\pgfqpoint{1.376cm}{-0.035cm}}
\pgfpathlineto{\pgfqpoint{1.376cm}{1.552cm}}
\pgfpathlineto{\pgfqpoint{0cm}{1.552cm}}
\pgfpathclose
\pgfusepath{clip}
\begin{pgfscope}
\begin{pgfscope}
\pgfsetdash{}{0cm}
\pgfsetlinewidth{0.818mm}
\pgfsetroundcap
\pgfsetroundjoin
\pgfsetmiterlimit{7.0}
\definecolor{eps2pgf_color}{gray}{0}\pgfsetstrokecolor{eps2pgf_color}\pgfsetfillcolor{eps2pgf_color}
\pgfpathmoveto{\pgfqpoint{0.117cm}{1.421cm}}
\pgfpathlineto{\pgfqpoint{0.682cm}{0.671cm}}
\pgfpathlineto{\pgfqpoint{1.246cm}{1.421cm}}
\pgfusepath{stroke}
\end{pgfscope}
\definecolor{eps2pgf_color}{gray}{0}\pgfsetstrokecolor{eps2pgf_color}\pgfsetfillcolor{eps2pgf_color}
\pgfpathmoveto{\pgfqpoint{0.273cm}{1.395cm}}
\pgfpathcurveto{\pgfqpoint{0.273cm}{1.432cm}}{\pgfqpoint{0.259cm}{1.467cm}}{\pgfqpoint{0.233cm}{1.492cm}}
\pgfpathcurveto{\pgfqpoint{0.207cm}{1.518cm}}{\pgfqpoint{0.173cm}{1.532cm}}{\pgfqpoint{0.137cm}{1.532cm}}
\pgfpathcurveto{\pgfqpoint{0.1cm}{1.532cm}}{\pgfqpoint{0.066cm}{1.518cm}}{\pgfqpoint{0.04cm}{1.492cm}}
\pgfpathcurveto{\pgfqpoint{0.014cm}{1.467cm}}{\pgfqpoint{0cm}{1.432cm}}{\pgfqpoint{0cm}{1.395cm}}
\pgfpathcurveto{\pgfqpoint{0cm}{1.359cm}}{\pgfqpoint{0.014cm}{1.324cm}}{\pgfqpoint{0.04cm}{1.299cm}}
\pgfpathcurveto{\pgfqpoint{0.066cm}{1.273cm}}{\pgfqpoint{0.1cm}{1.258cm}}{\pgfqpoint{0.137cm}{1.258cm}}
\pgfpathcurveto{\pgfqpoint{0.173cm}{1.258cm}}{\pgfqpoint{0.207cm}{1.273cm}}{\pgfqpoint{0.233cm}{1.299cm}}
\pgfpathcurveto{\pgfqpoint{0.259cm}{1.324cm}}{\pgfqpoint{0.273cm}{1.359cm}}{\pgfqpoint{0.273cm}{1.395cm}}
\pgfusepath{fill}
\begin{pgfscope}
\pgfsetdash{}{0cm}
\pgfsetlinewidth{0.818mm}
\pgfsetmiterlimit{7.0}
\pgfpathmoveto{\pgfqpoint{0.682cm}{0.671cm}}
\pgfpathlineto{\pgfqpoint{0.679cm}{1.418cm}}
\pgfusepath{stroke}
\end{pgfscope}
\pgfpathmoveto{\pgfqpoint{0.815cm}{1.399cm}}
\pgfpathcurveto{\pgfqpoint{0.815cm}{1.435cm}}{\pgfqpoint{0.801cm}{1.47cm}}{\pgfqpoint{0.775cm}{1.496cm}}
\pgfpathcurveto{\pgfqpoint{0.75cm}{1.521cm}}{\pgfqpoint{0.715cm}{1.536cm}}{\pgfqpoint{0.679cm}{1.536cm}}
\pgfpathcurveto{\pgfqpoint{0.643cm}{1.536cm}}{\pgfqpoint{0.608cm}{1.521cm}}{\pgfqpoint{0.582cm}{1.496cm}}
\pgfpathcurveto{\pgfqpoint{0.557cm}{1.47cm}}{\pgfqpoint{0.542cm}{1.435cm}}{\pgfqpoint{0.542cm}{1.399cm}}
\pgfpathcurveto{\pgfqpoint{0.542cm}{1.363cm}}{\pgfqpoint{0.557cm}{1.328cm}}{\pgfqpoint{0.582cm}{1.302cm}}
\pgfpathcurveto{\pgfqpoint{0.608cm}{1.276cm}}{\pgfqpoint{0.643cm}{1.262cm}}{\pgfqpoint{0.679cm}{1.262cm}}
\pgfpathcurveto{\pgfqpoint{0.715cm}{1.262cm}}{\pgfqpoint{0.75cm}{1.276cm}}{\pgfqpoint{0.775cm}{1.302cm}}
\pgfpathcurveto{\pgfqpoint{0.801cm}{1.328cm}}{\pgfqpoint{0.815cm}{1.363cm}}{\pgfqpoint{0.815cm}{1.399cm}}
\pgfusepath{fill}
\pgfpathmoveto{\pgfqpoint{1.345cm}{1.371cm}}
\pgfpathcurveto{\pgfqpoint{1.345cm}{1.408cm}}{\pgfqpoint{1.331cm}{1.442cm}}{\pgfqpoint{1.305cm}{1.468cm}}
\pgfpathcurveto{\pgfqpoint{1.28cm}{1.494cm}}{\pgfqpoint{1.245cm}{1.508cm}}{\pgfqpoint{1.209cm}{1.508cm}}
\pgfpathcurveto{\pgfqpoint{1.172cm}{1.508cm}}{\pgfqpoint{1.138cm}{1.494cm}}{\pgfqpoint{1.112cm}{1.468cm}}
\pgfpathcurveto{\pgfqpoint{1.087cm}{1.442cm}}{\pgfqpoint{1.072cm}{1.408cm}}{\pgfqpoint{1.072cm}{1.371cm}}
\pgfpathcurveto{\pgfqpoint{1.072cm}{1.335cm}}{\pgfqpoint{1.087cm}{1.3cm}}{\pgfqpoint{1.112cm}{1.274cm}}
\pgfpathcurveto{\pgfqpoint{1.138cm}{1.249cm}}{\pgfqpoint{1.172cm}{1.234cm}}{\pgfqpoint{1.209cm}{1.234cm}}
\pgfpathcurveto{\pgfqpoint{1.245cm}{1.234cm}}{\pgfqpoint{1.28cm}{1.249cm}}{\pgfqpoint{1.305cm}{1.274cm}}
\pgfpathcurveto{\pgfqpoint{1.331cm}{1.3cm}}{\pgfqpoint{1.345cm}{1.335cm}}{\pgfqpoint{1.345cm}{1.371cm}}
\pgfusepath{fill}
\begin{pgfscope}
\pgfsetdash{}{0cm}
\pgfsetlinewidth{0.818mm}
\pgfsetroundcap
\pgfsetmiterlimit{4.0}
\pgfpathmoveto{\pgfqpoint{0.682cm}{0.671cm}}
\pgfpathlineto{\pgfqpoint{0.682cm}{0.042cm}}
\pgfusepath{stroke}
\end{pgfscope}
\end{pgfscope}
\end{pgfscope}
\end{pgfscope}
\end{tikzpicture}}}$ were already defined in \eqref{eq:X}, \eqref{eq:X2X3} and \eqref{eq:Xt31}.
As the next step we provide further details and construct additional stochastic objects needed in the sequel. All the distributions on
$\Lambda_{M, \varepsilon}$ are extended periodically to the full lattice
$\Lambda_{\varepsilon}$. Then $X_{M, \varepsilon}^{\!\resizebox{0.6em}{!}{
\begin{tikzpicture}
\pgfpathmoveto{\pgfqpoint{0cm}{-0.035cm}}
\pgfpathlineto{\pgfqpoint{1.376cm}{-0.035cm}}
\pgfpathlineto{\pgfqpoint{1.376cm}{1.552cm}}
\pgfpathlineto{\pgfqpoint{0cm}{1.552cm}}
\pgfpathclose
\pgfusepath{clip}
\begin{pgfscope}
\begin{pgfscope}
\pgfpathmoveto{\pgfqpoint{0cm}{-0.035cm}}
\pgfpathlineto{\pgfqpoint{1.376cm}{-0.035cm}}
\pgfpathlineto{\pgfqpoint{1.376cm}{1.552cm}}
\pgfpathlineto{\pgfqpoint{0cm}{1.552cm}}
\pgfpathclose
\pgfusepath{clip}
\begin{pgfscope}
\begin{pgfscope}
\pgfsetdash{}{0cm}
\pgfsetlinewidth{0.818mm}
\pgfsetroundcap
\pgfsetroundjoin
\pgfsetmiterlimit{7.0}
\definecolor{eps2pgf_color}{gray}{0}\pgfsetstrokecolor{eps2pgf_color}\pgfsetfillcolor{eps2pgf_color}
\pgfpathmoveto{\pgfqpoint{0.117cm}{1.421cm}}
\pgfpathlineto{\pgfqpoint{0.682cm}{0.671cm}}
\pgfpathlineto{\pgfqpoint{1.246cm}{1.421cm}}
\pgfusepath{stroke}
\end{pgfscope}
\definecolor{eps2pgf_color}{gray}{0}\pgfsetstrokecolor{eps2pgf_color}\pgfsetfillcolor{eps2pgf_color}
\pgfpathmoveto{\pgfqpoint{0.273cm}{1.395cm}}
\pgfpathcurveto{\pgfqpoint{0.273cm}{1.432cm}}{\pgfqpoint{0.259cm}{1.467cm}}{\pgfqpoint{0.233cm}{1.492cm}}
\pgfpathcurveto{\pgfqpoint{0.207cm}{1.518cm}}{\pgfqpoint{0.173cm}{1.532cm}}{\pgfqpoint{0.137cm}{1.532cm}}
\pgfpathcurveto{\pgfqpoint{0.1cm}{1.532cm}}{\pgfqpoint{0.066cm}{1.518cm}}{\pgfqpoint{0.04cm}{1.492cm}}
\pgfpathcurveto{\pgfqpoint{0.014cm}{1.467cm}}{\pgfqpoint{0cm}{1.432cm}}{\pgfqpoint{0cm}{1.395cm}}
\pgfpathcurveto{\pgfqpoint{0cm}{1.359cm}}{\pgfqpoint{0.014cm}{1.324cm}}{\pgfqpoint{0.04cm}{1.299cm}}
\pgfpathcurveto{\pgfqpoint{0.066cm}{1.273cm}}{\pgfqpoint{0.1cm}{1.258cm}}{\pgfqpoint{0.137cm}{1.258cm}}
\pgfpathcurveto{\pgfqpoint{0.173cm}{1.258cm}}{\pgfqpoint{0.207cm}{1.273cm}}{\pgfqpoint{0.233cm}{1.299cm}}
\pgfpathcurveto{\pgfqpoint{0.259cm}{1.324cm}}{\pgfqpoint{0.273cm}{1.359cm}}{\pgfqpoint{0.273cm}{1.395cm}}
\pgfusepath{fill}
\begin{pgfscope}
\pgfsetdash{}{0cm}
\pgfsetlinewidth{0.818mm}
\pgfsetmiterlimit{7.0}
\pgfpathmoveto{\pgfqpoint{0.682cm}{0.671cm}}
\pgfpathlineto{\pgfqpoint{0.679cm}{1.418cm}}
\pgfusepath{stroke}
\end{pgfscope}
\pgfpathmoveto{\pgfqpoint{0.815cm}{1.399cm}}
\pgfpathcurveto{\pgfqpoint{0.815cm}{1.435cm}}{\pgfqpoint{0.801cm}{1.47cm}}{\pgfqpoint{0.775cm}{1.496cm}}
\pgfpathcurveto{\pgfqpoint{0.75cm}{1.521cm}}{\pgfqpoint{0.715cm}{1.536cm}}{\pgfqpoint{0.679cm}{1.536cm}}
\pgfpathcurveto{\pgfqpoint{0.643cm}{1.536cm}}{\pgfqpoint{0.608cm}{1.521cm}}{\pgfqpoint{0.582cm}{1.496cm}}
\pgfpathcurveto{\pgfqpoint{0.557cm}{1.47cm}}{\pgfqpoint{0.542cm}{1.435cm}}{\pgfqpoint{0.542cm}{1.399cm}}
\pgfpathcurveto{\pgfqpoint{0.542cm}{1.363cm}}{\pgfqpoint{0.557cm}{1.328cm}}{\pgfqpoint{0.582cm}{1.302cm}}
\pgfpathcurveto{\pgfqpoint{0.608cm}{1.276cm}}{\pgfqpoint{0.643cm}{1.262cm}}{\pgfqpoint{0.679cm}{1.262cm}}
\pgfpathcurveto{\pgfqpoint{0.715cm}{1.262cm}}{\pgfqpoint{0.75cm}{1.276cm}}{\pgfqpoint{0.775cm}{1.302cm}}
\pgfpathcurveto{\pgfqpoint{0.801cm}{1.328cm}}{\pgfqpoint{0.815cm}{1.363cm}}{\pgfqpoint{0.815cm}{1.399cm}}
\pgfusepath{fill}
\pgfpathmoveto{\pgfqpoint{1.345cm}{1.371cm}}
\pgfpathcurveto{\pgfqpoint{1.345cm}{1.408cm}}{\pgfqpoint{1.331cm}{1.442cm}}{\pgfqpoint{1.305cm}{1.468cm}}
\pgfpathcurveto{\pgfqpoint{1.28cm}{1.494cm}}{\pgfqpoint{1.245cm}{1.508cm}}{\pgfqpoint{1.209cm}{1.508cm}}
\pgfpathcurveto{\pgfqpoint{1.172cm}{1.508cm}}{\pgfqpoint{1.138cm}{1.494cm}}{\pgfqpoint{1.112cm}{1.468cm}}
\pgfpathcurveto{\pgfqpoint{1.087cm}{1.442cm}}{\pgfqpoint{1.072cm}{1.408cm}}{\pgfqpoint{1.072cm}{1.371cm}}
\pgfpathcurveto{\pgfqpoint{1.072cm}{1.335cm}}{\pgfqpoint{1.087cm}{1.3cm}}{\pgfqpoint{1.112cm}{1.274cm}}
\pgfpathcurveto{\pgfqpoint{1.138cm}{1.249cm}}{\pgfqpoint{1.172cm}{1.234cm}}{\pgfqpoint{1.209cm}{1.234cm}}
\pgfpathcurveto{\pgfqpoint{1.245cm}{1.234cm}}{\pgfqpoint{1.28cm}{1.249cm}}{\pgfqpoint{1.305cm}{1.274cm}}
\pgfpathcurveto{\pgfqpoint{1.331cm}{1.3cm}}{\pgfqpoint{1.345cm}{1.335cm}}{\pgfqpoint{1.345cm}{1.371cm}}
\pgfusepath{fill}
\begin{pgfscope}
\pgfsetdash{}{0cm}
\pgfsetlinewidth{0.818mm}
\pgfsetroundcap
\pgfsetmiterlimit{4.0}
\pgfpathmoveto{\pgfqpoint{0.682cm}{0.671cm}}
\pgfpathlineto{\pgfqpoint{0.682cm}{0.042cm}}
\pgfusepath{stroke}
\end{pgfscope}
\end{pgfscope}
\end{pgfscope}
\end{pgfscope}
\end{tikzpicture}}}$ which is a stationary solution to \eqref{eq:Xt31}
satisfies $X_{M, \varepsilon}^{\!\resizebox{0.6em}{!}{
\begin{tikzpicture}
\pgfpathmoveto{\pgfqpoint{0cm}{-0.035cm}}
\pgfpathlineto{\pgfqpoint{1.376cm}{-0.035cm}}
\pgfpathlineto{\pgfqpoint{1.376cm}{1.552cm}}
\pgfpathlineto{\pgfqpoint{0cm}{1.552cm}}
\pgfpathclose
\pgfusepath{clip}
\begin{pgfscope}
\begin{pgfscope}
\pgfpathmoveto{\pgfqpoint{0cm}{-0.035cm}}
\pgfpathlineto{\pgfqpoint{1.376cm}{-0.035cm}}
\pgfpathlineto{\pgfqpoint{1.376cm}{1.552cm}}
\pgfpathlineto{\pgfqpoint{0cm}{1.552cm}}
\pgfpathclose
\pgfusepath{clip}
\begin{pgfscope}
\begin{pgfscope}
\pgfsetdash{}{0cm}
\pgfsetlinewidth{0.818mm}
\pgfsetroundcap
\pgfsetroundjoin
\pgfsetmiterlimit{7.0}
\definecolor{eps2pgf_color}{gray}{0}\pgfsetstrokecolor{eps2pgf_color}\pgfsetfillcolor{eps2pgf_color}
\pgfpathmoveto{\pgfqpoint{0.117cm}{1.421cm}}
\pgfpathlineto{\pgfqpoint{0.682cm}{0.671cm}}
\pgfpathlineto{\pgfqpoint{1.246cm}{1.421cm}}
\pgfusepath{stroke}
\end{pgfscope}
\definecolor{eps2pgf_color}{gray}{0}\pgfsetstrokecolor{eps2pgf_color}\pgfsetfillcolor{eps2pgf_color}
\pgfpathmoveto{\pgfqpoint{0.273cm}{1.395cm}}
\pgfpathcurveto{\pgfqpoint{0.273cm}{1.432cm}}{\pgfqpoint{0.259cm}{1.467cm}}{\pgfqpoint{0.233cm}{1.492cm}}
\pgfpathcurveto{\pgfqpoint{0.207cm}{1.518cm}}{\pgfqpoint{0.173cm}{1.532cm}}{\pgfqpoint{0.137cm}{1.532cm}}
\pgfpathcurveto{\pgfqpoint{0.1cm}{1.532cm}}{\pgfqpoint{0.066cm}{1.518cm}}{\pgfqpoint{0.04cm}{1.492cm}}
\pgfpathcurveto{\pgfqpoint{0.014cm}{1.467cm}}{\pgfqpoint{0cm}{1.432cm}}{\pgfqpoint{0cm}{1.395cm}}
\pgfpathcurveto{\pgfqpoint{0cm}{1.359cm}}{\pgfqpoint{0.014cm}{1.324cm}}{\pgfqpoint{0.04cm}{1.299cm}}
\pgfpathcurveto{\pgfqpoint{0.066cm}{1.273cm}}{\pgfqpoint{0.1cm}{1.258cm}}{\pgfqpoint{0.137cm}{1.258cm}}
\pgfpathcurveto{\pgfqpoint{0.173cm}{1.258cm}}{\pgfqpoint{0.207cm}{1.273cm}}{\pgfqpoint{0.233cm}{1.299cm}}
\pgfpathcurveto{\pgfqpoint{0.259cm}{1.324cm}}{\pgfqpoint{0.273cm}{1.359cm}}{\pgfqpoint{0.273cm}{1.395cm}}
\pgfusepath{fill}
\begin{pgfscope}
\pgfsetdash{}{0cm}
\pgfsetlinewidth{0.818mm}
\pgfsetmiterlimit{7.0}
\pgfpathmoveto{\pgfqpoint{0.682cm}{0.671cm}}
\pgfpathlineto{\pgfqpoint{0.679cm}{1.418cm}}
\pgfusepath{stroke}
\end{pgfscope}
\pgfpathmoveto{\pgfqpoint{0.815cm}{1.399cm}}
\pgfpathcurveto{\pgfqpoint{0.815cm}{1.435cm}}{\pgfqpoint{0.801cm}{1.47cm}}{\pgfqpoint{0.775cm}{1.496cm}}
\pgfpathcurveto{\pgfqpoint{0.75cm}{1.521cm}}{\pgfqpoint{0.715cm}{1.536cm}}{\pgfqpoint{0.679cm}{1.536cm}}
\pgfpathcurveto{\pgfqpoint{0.643cm}{1.536cm}}{\pgfqpoint{0.608cm}{1.521cm}}{\pgfqpoint{0.582cm}{1.496cm}}
\pgfpathcurveto{\pgfqpoint{0.557cm}{1.47cm}}{\pgfqpoint{0.542cm}{1.435cm}}{\pgfqpoint{0.542cm}{1.399cm}}
\pgfpathcurveto{\pgfqpoint{0.542cm}{1.363cm}}{\pgfqpoint{0.557cm}{1.328cm}}{\pgfqpoint{0.582cm}{1.302cm}}
\pgfpathcurveto{\pgfqpoint{0.608cm}{1.276cm}}{\pgfqpoint{0.643cm}{1.262cm}}{\pgfqpoint{0.679cm}{1.262cm}}
\pgfpathcurveto{\pgfqpoint{0.715cm}{1.262cm}}{\pgfqpoint{0.75cm}{1.276cm}}{\pgfqpoint{0.775cm}{1.302cm}}
\pgfpathcurveto{\pgfqpoint{0.801cm}{1.328cm}}{\pgfqpoint{0.815cm}{1.363cm}}{\pgfqpoint{0.815cm}{1.399cm}}
\pgfusepath{fill}
\pgfpathmoveto{\pgfqpoint{1.345cm}{1.371cm}}
\pgfpathcurveto{\pgfqpoint{1.345cm}{1.408cm}}{\pgfqpoint{1.331cm}{1.442cm}}{\pgfqpoint{1.305cm}{1.468cm}}
\pgfpathcurveto{\pgfqpoint{1.28cm}{1.494cm}}{\pgfqpoint{1.245cm}{1.508cm}}{\pgfqpoint{1.209cm}{1.508cm}}
\pgfpathcurveto{\pgfqpoint{1.172cm}{1.508cm}}{\pgfqpoint{1.138cm}{1.494cm}}{\pgfqpoint{1.112cm}{1.468cm}}
\pgfpathcurveto{\pgfqpoint{1.087cm}{1.442cm}}{\pgfqpoint{1.072cm}{1.408cm}}{\pgfqpoint{1.072cm}{1.371cm}}
\pgfpathcurveto{\pgfqpoint{1.072cm}{1.335cm}}{\pgfqpoint{1.087cm}{1.3cm}}{\pgfqpoint{1.112cm}{1.274cm}}
\pgfpathcurveto{\pgfqpoint{1.138cm}{1.249cm}}{\pgfqpoint{1.172cm}{1.234cm}}{\pgfqpoint{1.209cm}{1.234cm}}
\pgfpathcurveto{\pgfqpoint{1.245cm}{1.234cm}}{\pgfqpoint{1.28cm}{1.249cm}}{\pgfqpoint{1.305cm}{1.274cm}}
\pgfpathcurveto{\pgfqpoint{1.331cm}{1.3cm}}{\pgfqpoint{1.345cm}{1.335cm}}{\pgfqpoint{1.345cm}{1.371cm}}
\pgfusepath{fill}
\begin{pgfscope}
\pgfsetdash{}{0cm}
\pgfsetlinewidth{0.818mm}
\pgfsetroundcap
\pgfsetmiterlimit{4.0}
\pgfpathmoveto{\pgfqpoint{0.682cm}{0.671cm}}
\pgfpathlineto{\pgfqpoint{0.682cm}{0.042cm}}
\pgfusepath{stroke}
\end{pgfscope}
\end{pgfscope}
\end{pgfscope}
\end{pgfscope}
\end{tikzpicture}}}(t) =
P^{\varepsilon}_{t}X_{M, \varepsilon}^{\!\resizebox{0.6em}{!}{
\begin{tikzpicture}
\pgfpathmoveto{\pgfqpoint{0cm}{-0.035cm}}
\pgfpathlineto{\pgfqpoint{1.376cm}{-0.035cm}}
\pgfpathlineto{\pgfqpoint{1.376cm}{1.552cm}}
\pgfpathlineto{\pgfqpoint{0cm}{1.552cm}}
\pgfpathclose
\pgfusepath{clip}
\begin{pgfscope}
\begin{pgfscope}
\pgfpathmoveto{\pgfqpoint{0cm}{-0.035cm}}
\pgfpathlineto{\pgfqpoint{1.376cm}{-0.035cm}}
\pgfpathlineto{\pgfqpoint{1.376cm}{1.552cm}}
\pgfpathlineto{\pgfqpoint{0cm}{1.552cm}}
\pgfpathclose
\pgfusepath{clip}
\begin{pgfscope}
\begin{pgfscope}
\pgfsetdash{}{0cm}
\pgfsetlinewidth{0.818mm}
\pgfsetroundcap
\pgfsetroundjoin
\pgfsetmiterlimit{7.0}
\definecolor{eps2pgf_color}{gray}{0}\pgfsetstrokecolor{eps2pgf_color}\pgfsetfillcolor{eps2pgf_color}
\pgfpathmoveto{\pgfqpoint{0.117cm}{1.421cm}}
\pgfpathlineto{\pgfqpoint{0.682cm}{0.671cm}}
\pgfpathlineto{\pgfqpoint{1.246cm}{1.421cm}}
\pgfusepath{stroke}
\end{pgfscope}
\definecolor{eps2pgf_color}{gray}{0}\pgfsetstrokecolor{eps2pgf_color}\pgfsetfillcolor{eps2pgf_color}
\pgfpathmoveto{\pgfqpoint{0.273cm}{1.395cm}}
\pgfpathcurveto{\pgfqpoint{0.273cm}{1.432cm}}{\pgfqpoint{0.259cm}{1.467cm}}{\pgfqpoint{0.233cm}{1.492cm}}
\pgfpathcurveto{\pgfqpoint{0.207cm}{1.518cm}}{\pgfqpoint{0.173cm}{1.532cm}}{\pgfqpoint{0.137cm}{1.532cm}}
\pgfpathcurveto{\pgfqpoint{0.1cm}{1.532cm}}{\pgfqpoint{0.066cm}{1.518cm}}{\pgfqpoint{0.04cm}{1.492cm}}
\pgfpathcurveto{\pgfqpoint{0.014cm}{1.467cm}}{\pgfqpoint{0cm}{1.432cm}}{\pgfqpoint{0cm}{1.395cm}}
\pgfpathcurveto{\pgfqpoint{0cm}{1.359cm}}{\pgfqpoint{0.014cm}{1.324cm}}{\pgfqpoint{0.04cm}{1.299cm}}
\pgfpathcurveto{\pgfqpoint{0.066cm}{1.273cm}}{\pgfqpoint{0.1cm}{1.258cm}}{\pgfqpoint{0.137cm}{1.258cm}}
\pgfpathcurveto{\pgfqpoint{0.173cm}{1.258cm}}{\pgfqpoint{0.207cm}{1.273cm}}{\pgfqpoint{0.233cm}{1.299cm}}
\pgfpathcurveto{\pgfqpoint{0.259cm}{1.324cm}}{\pgfqpoint{0.273cm}{1.359cm}}{\pgfqpoint{0.273cm}{1.395cm}}
\pgfusepath{fill}
\begin{pgfscope}
\pgfsetdash{}{0cm}
\pgfsetlinewidth{0.818mm}
\pgfsetmiterlimit{7.0}
\pgfpathmoveto{\pgfqpoint{0.682cm}{0.671cm}}
\pgfpathlineto{\pgfqpoint{0.679cm}{1.418cm}}
\pgfusepath{stroke}
\end{pgfscope}
\pgfpathmoveto{\pgfqpoint{0.815cm}{1.399cm}}
\pgfpathcurveto{\pgfqpoint{0.815cm}{1.435cm}}{\pgfqpoint{0.801cm}{1.47cm}}{\pgfqpoint{0.775cm}{1.496cm}}
\pgfpathcurveto{\pgfqpoint{0.75cm}{1.521cm}}{\pgfqpoint{0.715cm}{1.536cm}}{\pgfqpoint{0.679cm}{1.536cm}}
\pgfpathcurveto{\pgfqpoint{0.643cm}{1.536cm}}{\pgfqpoint{0.608cm}{1.521cm}}{\pgfqpoint{0.582cm}{1.496cm}}
\pgfpathcurveto{\pgfqpoint{0.557cm}{1.47cm}}{\pgfqpoint{0.542cm}{1.435cm}}{\pgfqpoint{0.542cm}{1.399cm}}
\pgfpathcurveto{\pgfqpoint{0.542cm}{1.363cm}}{\pgfqpoint{0.557cm}{1.328cm}}{\pgfqpoint{0.582cm}{1.302cm}}
\pgfpathcurveto{\pgfqpoint{0.608cm}{1.276cm}}{\pgfqpoint{0.643cm}{1.262cm}}{\pgfqpoint{0.679cm}{1.262cm}}
\pgfpathcurveto{\pgfqpoint{0.715cm}{1.262cm}}{\pgfqpoint{0.75cm}{1.276cm}}{\pgfqpoint{0.775cm}{1.302cm}}
\pgfpathcurveto{\pgfqpoint{0.801cm}{1.328cm}}{\pgfqpoint{0.815cm}{1.363cm}}{\pgfqpoint{0.815cm}{1.399cm}}
\pgfusepath{fill}
\pgfpathmoveto{\pgfqpoint{1.345cm}{1.371cm}}
\pgfpathcurveto{\pgfqpoint{1.345cm}{1.408cm}}{\pgfqpoint{1.331cm}{1.442cm}}{\pgfqpoint{1.305cm}{1.468cm}}
\pgfpathcurveto{\pgfqpoint{1.28cm}{1.494cm}}{\pgfqpoint{1.245cm}{1.508cm}}{\pgfqpoint{1.209cm}{1.508cm}}
\pgfpathcurveto{\pgfqpoint{1.172cm}{1.508cm}}{\pgfqpoint{1.138cm}{1.494cm}}{\pgfqpoint{1.112cm}{1.468cm}}
\pgfpathcurveto{\pgfqpoint{1.087cm}{1.442cm}}{\pgfqpoint{1.072cm}{1.408cm}}{\pgfqpoint{1.072cm}{1.371cm}}
\pgfpathcurveto{\pgfqpoint{1.072cm}{1.335cm}}{\pgfqpoint{1.087cm}{1.3cm}}{\pgfqpoint{1.112cm}{1.274cm}}
\pgfpathcurveto{\pgfqpoint{1.138cm}{1.249cm}}{\pgfqpoint{1.172cm}{1.234cm}}{\pgfqpoint{1.209cm}{1.234cm}}
\pgfpathcurveto{\pgfqpoint{1.245cm}{1.234cm}}{\pgfqpoint{1.28cm}{1.249cm}}{\pgfqpoint{1.305cm}{1.274cm}}
\pgfpathcurveto{\pgfqpoint{1.331cm}{1.3cm}}{\pgfqpoint{1.345cm}{1.335cm}}{\pgfqpoint{1.345cm}{1.371cm}}
\pgfusepath{fill}
\begin{pgfscope}
\pgfsetdash{}{0cm}
\pgfsetlinewidth{0.818mm}
\pgfsetroundcap
\pgfsetmiterlimit{4.0}
\pgfpathmoveto{\pgfqpoint{0.682cm}{0.671cm}}
\pgfpathlineto{\pgfqpoint{0.682cm}{0.042cm}}
\pgfusepath{stroke}
\end{pgfscope}
\end{pgfscope}
\end{pgfscope}
\end{pgfscope}
\end{tikzpicture}}} (0) + \LL_{\varepsilon}^{- 1} \llbracket X_{M,
\varepsilon}^3 \rrbracket$ with $X_{M, \varepsilon}^{\!\resizebox{0.6em}{!}{
\begin{tikzpicture}
\pgfpathmoveto{\pgfqpoint{0cm}{-0.035cm}}
\pgfpathlineto{\pgfqpoint{1.376cm}{-0.035cm}}
\pgfpathlineto{\pgfqpoint{1.376cm}{1.552cm}}
\pgfpathlineto{\pgfqpoint{0cm}{1.552cm}}
\pgfpathclose
\pgfusepath{clip}
\begin{pgfscope}
\begin{pgfscope}
\pgfpathmoveto{\pgfqpoint{0cm}{-0.035cm}}
\pgfpathlineto{\pgfqpoint{1.376cm}{-0.035cm}}
\pgfpathlineto{\pgfqpoint{1.376cm}{1.552cm}}
\pgfpathlineto{\pgfqpoint{0cm}{1.552cm}}
\pgfpathclose
\pgfusepath{clip}
\begin{pgfscope}
\begin{pgfscope}
\pgfsetdash{}{0cm}
\pgfsetlinewidth{0.818mm}
\pgfsetroundcap
\pgfsetroundjoin
\pgfsetmiterlimit{7.0}
\definecolor{eps2pgf_color}{gray}{0}\pgfsetstrokecolor{eps2pgf_color}\pgfsetfillcolor{eps2pgf_color}
\pgfpathmoveto{\pgfqpoint{0.117cm}{1.421cm}}
\pgfpathlineto{\pgfqpoint{0.682cm}{0.671cm}}
\pgfpathlineto{\pgfqpoint{1.246cm}{1.421cm}}
\pgfusepath{stroke}
\end{pgfscope}
\definecolor{eps2pgf_color}{gray}{0}\pgfsetstrokecolor{eps2pgf_color}\pgfsetfillcolor{eps2pgf_color}
\pgfpathmoveto{\pgfqpoint{0.273cm}{1.395cm}}
\pgfpathcurveto{\pgfqpoint{0.273cm}{1.432cm}}{\pgfqpoint{0.259cm}{1.467cm}}{\pgfqpoint{0.233cm}{1.492cm}}
\pgfpathcurveto{\pgfqpoint{0.207cm}{1.518cm}}{\pgfqpoint{0.173cm}{1.532cm}}{\pgfqpoint{0.137cm}{1.532cm}}
\pgfpathcurveto{\pgfqpoint{0.1cm}{1.532cm}}{\pgfqpoint{0.066cm}{1.518cm}}{\pgfqpoint{0.04cm}{1.492cm}}
\pgfpathcurveto{\pgfqpoint{0.014cm}{1.467cm}}{\pgfqpoint{0cm}{1.432cm}}{\pgfqpoint{0cm}{1.395cm}}
\pgfpathcurveto{\pgfqpoint{0cm}{1.359cm}}{\pgfqpoint{0.014cm}{1.324cm}}{\pgfqpoint{0.04cm}{1.299cm}}
\pgfpathcurveto{\pgfqpoint{0.066cm}{1.273cm}}{\pgfqpoint{0.1cm}{1.258cm}}{\pgfqpoint{0.137cm}{1.258cm}}
\pgfpathcurveto{\pgfqpoint{0.173cm}{1.258cm}}{\pgfqpoint{0.207cm}{1.273cm}}{\pgfqpoint{0.233cm}{1.299cm}}
\pgfpathcurveto{\pgfqpoint{0.259cm}{1.324cm}}{\pgfqpoint{0.273cm}{1.359cm}}{\pgfqpoint{0.273cm}{1.395cm}}
\pgfusepath{fill}
\begin{pgfscope}
\pgfsetdash{}{0cm}
\pgfsetlinewidth{0.818mm}
\pgfsetmiterlimit{7.0}
\pgfpathmoveto{\pgfqpoint{0.682cm}{0.671cm}}
\pgfpathlineto{\pgfqpoint{0.679cm}{1.418cm}}
\pgfusepath{stroke}
\end{pgfscope}
\pgfpathmoveto{\pgfqpoint{0.815cm}{1.399cm}}
\pgfpathcurveto{\pgfqpoint{0.815cm}{1.435cm}}{\pgfqpoint{0.801cm}{1.47cm}}{\pgfqpoint{0.775cm}{1.496cm}}
\pgfpathcurveto{\pgfqpoint{0.75cm}{1.521cm}}{\pgfqpoint{0.715cm}{1.536cm}}{\pgfqpoint{0.679cm}{1.536cm}}
\pgfpathcurveto{\pgfqpoint{0.643cm}{1.536cm}}{\pgfqpoint{0.608cm}{1.521cm}}{\pgfqpoint{0.582cm}{1.496cm}}
\pgfpathcurveto{\pgfqpoint{0.557cm}{1.47cm}}{\pgfqpoint{0.542cm}{1.435cm}}{\pgfqpoint{0.542cm}{1.399cm}}
\pgfpathcurveto{\pgfqpoint{0.542cm}{1.363cm}}{\pgfqpoint{0.557cm}{1.328cm}}{\pgfqpoint{0.582cm}{1.302cm}}
\pgfpathcurveto{\pgfqpoint{0.608cm}{1.276cm}}{\pgfqpoint{0.643cm}{1.262cm}}{\pgfqpoint{0.679cm}{1.262cm}}
\pgfpathcurveto{\pgfqpoint{0.715cm}{1.262cm}}{\pgfqpoint{0.75cm}{1.276cm}}{\pgfqpoint{0.775cm}{1.302cm}}
\pgfpathcurveto{\pgfqpoint{0.801cm}{1.328cm}}{\pgfqpoint{0.815cm}{1.363cm}}{\pgfqpoint{0.815cm}{1.399cm}}
\pgfusepath{fill}
\pgfpathmoveto{\pgfqpoint{1.345cm}{1.371cm}}
\pgfpathcurveto{\pgfqpoint{1.345cm}{1.408cm}}{\pgfqpoint{1.331cm}{1.442cm}}{\pgfqpoint{1.305cm}{1.468cm}}
\pgfpathcurveto{\pgfqpoint{1.28cm}{1.494cm}}{\pgfqpoint{1.245cm}{1.508cm}}{\pgfqpoint{1.209cm}{1.508cm}}
\pgfpathcurveto{\pgfqpoint{1.172cm}{1.508cm}}{\pgfqpoint{1.138cm}{1.494cm}}{\pgfqpoint{1.112cm}{1.468cm}}
\pgfpathcurveto{\pgfqpoint{1.087cm}{1.442cm}}{\pgfqpoint{1.072cm}{1.408cm}}{\pgfqpoint{1.072cm}{1.371cm}}
\pgfpathcurveto{\pgfqpoint{1.072cm}{1.335cm}}{\pgfqpoint{1.087cm}{1.3cm}}{\pgfqpoint{1.112cm}{1.274cm}}
\pgfpathcurveto{\pgfqpoint{1.138cm}{1.249cm}}{\pgfqpoint{1.172cm}{1.234cm}}{\pgfqpoint{1.209cm}{1.234cm}}
\pgfpathcurveto{\pgfqpoint{1.245cm}{1.234cm}}{\pgfqpoint{1.28cm}{1.249cm}}{\pgfqpoint{1.305cm}{1.274cm}}
\pgfpathcurveto{\pgfqpoint{1.331cm}{1.3cm}}{\pgfqpoint{1.345cm}{1.335cm}}{\pgfqpoint{1.345cm}{1.371cm}}
\pgfusepath{fill}
\begin{pgfscope}
\pgfsetdash{}{0cm}
\pgfsetlinewidth{0.818mm}
\pgfsetroundcap
\pgfsetmiterlimit{4.0}
\pgfpathmoveto{\pgfqpoint{0.682cm}{0.671cm}}
\pgfpathlineto{\pgfqpoint{0.682cm}{0.042cm}}
\pgfusepath{stroke}
\end{pgfscope}
\end{pgfscope}
\end{pgfscope}
\end{pgfscope}
\end{tikzpicture}}} (0) = \int_{-
\infty}^0 P^{\varepsilon}_{- s} \llbracket X_{M, \varepsilon}^3 \rrbracket (s)
\mathd s$, where $P^{\varepsilon}_t$ denotes the semigroup generated by
$-\Q_{\varepsilon}$ on $\Lambda_{\varepsilon}$. Then  for every
$\kappa, \sigma > 0$ and some $\beta > 0$ small
\[ \| X_{M, \varepsilon}^{\!\resizebox{0.6em}{!}{
\begin{tikzpicture}
\pgfpathmoveto{\pgfqpoint{0cm}{-0.035cm}}
\pgfpathlineto{\pgfqpoint{1.376cm}{-0.035cm}}
\pgfpathlineto{\pgfqpoint{1.376cm}{1.552cm}}
\pgfpathlineto{\pgfqpoint{0cm}{1.552cm}}
\pgfpathclose
\pgfusepath{clip}
\begin{pgfscope}
\begin{pgfscope}
\pgfpathmoveto{\pgfqpoint{0cm}{-0.035cm}}
\pgfpathlineto{\pgfqpoint{1.376cm}{-0.035cm}}
\pgfpathlineto{\pgfqpoint{1.376cm}{1.552cm}}
\pgfpathlineto{\pgfqpoint{0cm}{1.552cm}}
\pgfpathclose
\pgfusepath{clip}
\begin{pgfscope}
\begin{pgfscope}
\pgfsetdash{}{0cm}
\pgfsetlinewidth{0.818mm}
\pgfsetroundcap
\pgfsetroundjoin
\pgfsetmiterlimit{7.0}
\definecolor{eps2pgf_color}{gray}{0}\pgfsetstrokecolor{eps2pgf_color}\pgfsetfillcolor{eps2pgf_color}
\pgfpathmoveto{\pgfqpoint{0.117cm}{1.421cm}}
\pgfpathlineto{\pgfqpoint{0.682cm}{0.671cm}}
\pgfpathlineto{\pgfqpoint{1.246cm}{1.421cm}}
\pgfusepath{stroke}
\end{pgfscope}
\definecolor{eps2pgf_color}{gray}{0}\pgfsetstrokecolor{eps2pgf_color}\pgfsetfillcolor{eps2pgf_color}
\pgfpathmoveto{\pgfqpoint{0.273cm}{1.395cm}}
\pgfpathcurveto{\pgfqpoint{0.273cm}{1.432cm}}{\pgfqpoint{0.259cm}{1.467cm}}{\pgfqpoint{0.233cm}{1.492cm}}
\pgfpathcurveto{\pgfqpoint{0.207cm}{1.518cm}}{\pgfqpoint{0.173cm}{1.532cm}}{\pgfqpoint{0.137cm}{1.532cm}}
\pgfpathcurveto{\pgfqpoint{0.1cm}{1.532cm}}{\pgfqpoint{0.066cm}{1.518cm}}{\pgfqpoint{0.04cm}{1.492cm}}
\pgfpathcurveto{\pgfqpoint{0.014cm}{1.467cm}}{\pgfqpoint{0cm}{1.432cm}}{\pgfqpoint{0cm}{1.395cm}}
\pgfpathcurveto{\pgfqpoint{0cm}{1.359cm}}{\pgfqpoint{0.014cm}{1.324cm}}{\pgfqpoint{0.04cm}{1.299cm}}
\pgfpathcurveto{\pgfqpoint{0.066cm}{1.273cm}}{\pgfqpoint{0.1cm}{1.258cm}}{\pgfqpoint{0.137cm}{1.258cm}}
\pgfpathcurveto{\pgfqpoint{0.173cm}{1.258cm}}{\pgfqpoint{0.207cm}{1.273cm}}{\pgfqpoint{0.233cm}{1.299cm}}
\pgfpathcurveto{\pgfqpoint{0.259cm}{1.324cm}}{\pgfqpoint{0.273cm}{1.359cm}}{\pgfqpoint{0.273cm}{1.395cm}}
\pgfusepath{fill}
\begin{pgfscope}
\pgfsetdash{}{0cm}
\pgfsetlinewidth{0.818mm}
\pgfsetmiterlimit{7.0}
\pgfpathmoveto{\pgfqpoint{0.682cm}{0.671cm}}
\pgfpathlineto{\pgfqpoint{0.679cm}{1.418cm}}
\pgfusepath{stroke}
\end{pgfscope}
\pgfpathmoveto{\pgfqpoint{0.815cm}{1.399cm}}
\pgfpathcurveto{\pgfqpoint{0.815cm}{1.435cm}}{\pgfqpoint{0.801cm}{1.47cm}}{\pgfqpoint{0.775cm}{1.496cm}}
\pgfpathcurveto{\pgfqpoint{0.75cm}{1.521cm}}{\pgfqpoint{0.715cm}{1.536cm}}{\pgfqpoint{0.679cm}{1.536cm}}
\pgfpathcurveto{\pgfqpoint{0.643cm}{1.536cm}}{\pgfqpoint{0.608cm}{1.521cm}}{\pgfqpoint{0.582cm}{1.496cm}}
\pgfpathcurveto{\pgfqpoint{0.557cm}{1.47cm}}{\pgfqpoint{0.542cm}{1.435cm}}{\pgfqpoint{0.542cm}{1.399cm}}
\pgfpathcurveto{\pgfqpoint{0.542cm}{1.363cm}}{\pgfqpoint{0.557cm}{1.328cm}}{\pgfqpoint{0.582cm}{1.302cm}}
\pgfpathcurveto{\pgfqpoint{0.608cm}{1.276cm}}{\pgfqpoint{0.643cm}{1.262cm}}{\pgfqpoint{0.679cm}{1.262cm}}
\pgfpathcurveto{\pgfqpoint{0.715cm}{1.262cm}}{\pgfqpoint{0.75cm}{1.276cm}}{\pgfqpoint{0.775cm}{1.302cm}}
\pgfpathcurveto{\pgfqpoint{0.801cm}{1.328cm}}{\pgfqpoint{0.815cm}{1.363cm}}{\pgfqpoint{0.815cm}{1.399cm}}
\pgfusepath{fill}
\pgfpathmoveto{\pgfqpoint{1.345cm}{1.371cm}}
\pgfpathcurveto{\pgfqpoint{1.345cm}{1.408cm}}{\pgfqpoint{1.331cm}{1.442cm}}{\pgfqpoint{1.305cm}{1.468cm}}
\pgfpathcurveto{\pgfqpoint{1.28cm}{1.494cm}}{\pgfqpoint{1.245cm}{1.508cm}}{\pgfqpoint{1.209cm}{1.508cm}}
\pgfpathcurveto{\pgfqpoint{1.172cm}{1.508cm}}{\pgfqpoint{1.138cm}{1.494cm}}{\pgfqpoint{1.112cm}{1.468cm}}
\pgfpathcurveto{\pgfqpoint{1.087cm}{1.442cm}}{\pgfqpoint{1.072cm}{1.408cm}}{\pgfqpoint{1.072cm}{1.371cm}}
\pgfpathcurveto{\pgfqpoint{1.072cm}{1.335cm}}{\pgfqpoint{1.087cm}{1.3cm}}{\pgfqpoint{1.112cm}{1.274cm}}
\pgfpathcurveto{\pgfqpoint{1.138cm}{1.249cm}}{\pgfqpoint{1.172cm}{1.234cm}}{\pgfqpoint{1.209cm}{1.234cm}}
\pgfpathcurveto{\pgfqpoint{1.245cm}{1.234cm}}{\pgfqpoint{1.28cm}{1.249cm}}{\pgfqpoint{1.305cm}{1.274cm}}
\pgfpathcurveto{\pgfqpoint{1.331cm}{1.3cm}}{\pgfqpoint{1.345cm}{1.335cm}}{\pgfqpoint{1.345cm}{1.371cm}}
\pgfusepath{fill}
\begin{pgfscope}
\pgfsetdash{}{0cm}
\pgfsetlinewidth{0.818mm}
\pgfsetroundcap
\pgfsetmiterlimit{4.0}
\pgfpathmoveto{\pgfqpoint{0.682cm}{0.671cm}}
\pgfpathlineto{\pgfqpoint{0.682cm}{0.042cm}}
\pgfusepath{stroke}
\end{pgfscope}
\end{pgfscope}
\end{pgfscope}
\end{pgfscope}
\end{tikzpicture}}} \|_{C_T \CC^{1 / 2 - \kappa,
   \varepsilon} (\rho^{\sigma})} + \| X_{M, \varepsilon}^{\!\resizebox{0.6em}{!}{
\begin{tikzpicture}
\pgfpathmoveto{\pgfqpoint{0cm}{-0.035cm}}
\pgfpathlineto{\pgfqpoint{1.376cm}{-0.035cm}}
\pgfpathlineto{\pgfqpoint{1.376cm}{1.552cm}}
\pgfpathlineto{\pgfqpoint{0cm}{1.552cm}}
\pgfpathclose
\pgfusepath{clip}
\begin{pgfscope}
\begin{pgfscope}
\pgfpathmoveto{\pgfqpoint{0cm}{-0.035cm}}
\pgfpathlineto{\pgfqpoint{1.376cm}{-0.035cm}}
\pgfpathlineto{\pgfqpoint{1.376cm}{1.552cm}}
\pgfpathlineto{\pgfqpoint{0cm}{1.552cm}}
\pgfpathclose
\pgfusepath{clip}
\begin{pgfscope}
\begin{pgfscope}
\pgfsetdash{}{0cm}
\pgfsetlinewidth{0.818mm}
\pgfsetroundcap
\pgfsetroundjoin
\pgfsetmiterlimit{7.0}
\definecolor{eps2pgf_color}{gray}{0}\pgfsetstrokecolor{eps2pgf_color}\pgfsetfillcolor{eps2pgf_color}
\pgfpathmoveto{\pgfqpoint{0.117cm}{1.421cm}}
\pgfpathlineto{\pgfqpoint{0.682cm}{0.671cm}}
\pgfpathlineto{\pgfqpoint{1.246cm}{1.421cm}}
\pgfusepath{stroke}
\end{pgfscope}
\definecolor{eps2pgf_color}{gray}{0}\pgfsetstrokecolor{eps2pgf_color}\pgfsetfillcolor{eps2pgf_color}
\pgfpathmoveto{\pgfqpoint{0.273cm}{1.395cm}}
\pgfpathcurveto{\pgfqpoint{0.273cm}{1.432cm}}{\pgfqpoint{0.259cm}{1.467cm}}{\pgfqpoint{0.233cm}{1.492cm}}
\pgfpathcurveto{\pgfqpoint{0.207cm}{1.518cm}}{\pgfqpoint{0.173cm}{1.532cm}}{\pgfqpoint{0.137cm}{1.532cm}}
\pgfpathcurveto{\pgfqpoint{0.1cm}{1.532cm}}{\pgfqpoint{0.066cm}{1.518cm}}{\pgfqpoint{0.04cm}{1.492cm}}
\pgfpathcurveto{\pgfqpoint{0.014cm}{1.467cm}}{\pgfqpoint{0cm}{1.432cm}}{\pgfqpoint{0cm}{1.395cm}}
\pgfpathcurveto{\pgfqpoint{0cm}{1.359cm}}{\pgfqpoint{0.014cm}{1.324cm}}{\pgfqpoint{0.04cm}{1.299cm}}
\pgfpathcurveto{\pgfqpoint{0.066cm}{1.273cm}}{\pgfqpoint{0.1cm}{1.258cm}}{\pgfqpoint{0.137cm}{1.258cm}}
\pgfpathcurveto{\pgfqpoint{0.173cm}{1.258cm}}{\pgfqpoint{0.207cm}{1.273cm}}{\pgfqpoint{0.233cm}{1.299cm}}
\pgfpathcurveto{\pgfqpoint{0.259cm}{1.324cm}}{\pgfqpoint{0.273cm}{1.359cm}}{\pgfqpoint{0.273cm}{1.395cm}}
\pgfusepath{fill}
\begin{pgfscope}
\pgfsetdash{}{0cm}
\pgfsetlinewidth{0.818mm}
\pgfsetmiterlimit{7.0}
\pgfpathmoveto{\pgfqpoint{0.682cm}{0.671cm}}
\pgfpathlineto{\pgfqpoint{0.679cm}{1.418cm}}
\pgfusepath{stroke}
\end{pgfscope}
\pgfpathmoveto{\pgfqpoint{0.815cm}{1.399cm}}
\pgfpathcurveto{\pgfqpoint{0.815cm}{1.435cm}}{\pgfqpoint{0.801cm}{1.47cm}}{\pgfqpoint{0.775cm}{1.496cm}}
\pgfpathcurveto{\pgfqpoint{0.75cm}{1.521cm}}{\pgfqpoint{0.715cm}{1.536cm}}{\pgfqpoint{0.679cm}{1.536cm}}
\pgfpathcurveto{\pgfqpoint{0.643cm}{1.536cm}}{\pgfqpoint{0.608cm}{1.521cm}}{\pgfqpoint{0.582cm}{1.496cm}}
\pgfpathcurveto{\pgfqpoint{0.557cm}{1.47cm}}{\pgfqpoint{0.542cm}{1.435cm}}{\pgfqpoint{0.542cm}{1.399cm}}
\pgfpathcurveto{\pgfqpoint{0.542cm}{1.363cm}}{\pgfqpoint{0.557cm}{1.328cm}}{\pgfqpoint{0.582cm}{1.302cm}}
\pgfpathcurveto{\pgfqpoint{0.608cm}{1.276cm}}{\pgfqpoint{0.643cm}{1.262cm}}{\pgfqpoint{0.679cm}{1.262cm}}
\pgfpathcurveto{\pgfqpoint{0.715cm}{1.262cm}}{\pgfqpoint{0.75cm}{1.276cm}}{\pgfqpoint{0.775cm}{1.302cm}}
\pgfpathcurveto{\pgfqpoint{0.801cm}{1.328cm}}{\pgfqpoint{0.815cm}{1.363cm}}{\pgfqpoint{0.815cm}{1.399cm}}
\pgfusepath{fill}
\pgfpathmoveto{\pgfqpoint{1.345cm}{1.371cm}}
\pgfpathcurveto{\pgfqpoint{1.345cm}{1.408cm}}{\pgfqpoint{1.331cm}{1.442cm}}{\pgfqpoint{1.305cm}{1.468cm}}
\pgfpathcurveto{\pgfqpoint{1.28cm}{1.494cm}}{\pgfqpoint{1.245cm}{1.508cm}}{\pgfqpoint{1.209cm}{1.508cm}}
\pgfpathcurveto{\pgfqpoint{1.172cm}{1.508cm}}{\pgfqpoint{1.138cm}{1.494cm}}{\pgfqpoint{1.112cm}{1.468cm}}
\pgfpathcurveto{\pgfqpoint{1.087cm}{1.442cm}}{\pgfqpoint{1.072cm}{1.408cm}}{\pgfqpoint{1.072cm}{1.371cm}}
\pgfpathcurveto{\pgfqpoint{1.072cm}{1.335cm}}{\pgfqpoint{1.087cm}{1.3cm}}{\pgfqpoint{1.112cm}{1.274cm}}
\pgfpathcurveto{\pgfqpoint{1.138cm}{1.249cm}}{\pgfqpoint{1.172cm}{1.234cm}}{\pgfqpoint{1.209cm}{1.234cm}}
\pgfpathcurveto{\pgfqpoint{1.245cm}{1.234cm}}{\pgfqpoint{1.28cm}{1.249cm}}{\pgfqpoint{1.305cm}{1.274cm}}
\pgfpathcurveto{\pgfqpoint{1.331cm}{1.3cm}}{\pgfqpoint{1.345cm}{1.335cm}}{\pgfqpoint{1.345cm}{1.371cm}}
\pgfusepath{fill}
\begin{pgfscope}
\pgfsetdash{}{0cm}
\pgfsetlinewidth{0.818mm}
\pgfsetroundcap
\pgfsetmiterlimit{4.0}
\pgfpathmoveto{\pgfqpoint{0.682cm}{0.671cm}}
\pgfpathlineto{\pgfqpoint{0.682cm}{0.042cm}}
\pgfusepath{stroke}
\end{pgfscope}
\end{pgfscope}
\end{pgfscope}
\end{pgfscope}
\end{tikzpicture}}}
   \|_{C_T^{\beta / 2} L^{\infty, \varepsilon} (\rho^{\sigma})} \lesssim
   1, \]
uniformly in $M, \varepsilon$ thanks to the presence of the weight. For details and further references see e.g. Section 3 in \cite{GH18}. Here and in the sequel, $T\in (0,\infty)$ denotes an arbitrary finite time horizon and $C_{T}$ and $C^{\beta/2}_{T}$ are shortcut notations for $C([0,T])$ and $C^{\beta/2}([0,T])$, respectively. Throughout
our analysis, we fix $\kappa, \beta > 0$ in the above estimate such that
$\beta \geqslant 3 \kappa$. This condition will be needed for the control of
a parabolic commutator  in Lemma
\ref{lemma:bounds-rhs1} below. On the other hand, the parameter $\sigma > 0$
varies from line to line and can be arbitrarily small.

As already discussed  in Section \ref{s:strat}, in particular after equation \eqref{eq:Xt31}, the usual decomposition $\varphi_{M,\varepsilon}=X_{M,\varepsilon}-\lambda X^{\!\resizebox{0.6em}{!}{
\begin{tikzpicture}
\pgfpathmoveto{\pgfqpoint{0cm}{-0.035cm}}
\pgfpathlineto{\pgfqpoint{1.376cm}{-0.035cm}}
\pgfpathlineto{\pgfqpoint{1.376cm}{1.552cm}}
\pgfpathlineto{\pgfqpoint{0cm}{1.552cm}}
\pgfpathclose
\pgfusepath{clip}
\begin{pgfscope}
\begin{pgfscope}
\pgfpathmoveto{\pgfqpoint{0cm}{-0.035cm}}
\pgfpathlineto{\pgfqpoint{1.376cm}{-0.035cm}}
\pgfpathlineto{\pgfqpoint{1.376cm}{1.552cm}}
\pgfpathlineto{\pgfqpoint{0cm}{1.552cm}}
\pgfpathclose
\pgfusepath{clip}
\begin{pgfscope}
\begin{pgfscope}
\pgfsetdash{}{0cm}
\pgfsetlinewidth{0.818mm}
\pgfsetroundcap
\pgfsetroundjoin
\pgfsetmiterlimit{7.0}
\definecolor{eps2pgf_color}{gray}{0}\pgfsetstrokecolor{eps2pgf_color}\pgfsetfillcolor{eps2pgf_color}
\pgfpathmoveto{\pgfqpoint{0.117cm}{1.421cm}}
\pgfpathlineto{\pgfqpoint{0.682cm}{0.671cm}}
\pgfpathlineto{\pgfqpoint{1.246cm}{1.421cm}}
\pgfusepath{stroke}
\end{pgfscope}
\definecolor{eps2pgf_color}{gray}{0}\pgfsetstrokecolor{eps2pgf_color}\pgfsetfillcolor{eps2pgf_color}
\pgfpathmoveto{\pgfqpoint{0.273cm}{1.395cm}}
\pgfpathcurveto{\pgfqpoint{0.273cm}{1.432cm}}{\pgfqpoint{0.259cm}{1.467cm}}{\pgfqpoint{0.233cm}{1.492cm}}
\pgfpathcurveto{\pgfqpoint{0.207cm}{1.518cm}}{\pgfqpoint{0.173cm}{1.532cm}}{\pgfqpoint{0.137cm}{1.532cm}}
\pgfpathcurveto{\pgfqpoint{0.1cm}{1.532cm}}{\pgfqpoint{0.066cm}{1.518cm}}{\pgfqpoint{0.04cm}{1.492cm}}
\pgfpathcurveto{\pgfqpoint{0.014cm}{1.467cm}}{\pgfqpoint{0cm}{1.432cm}}{\pgfqpoint{0cm}{1.395cm}}
\pgfpathcurveto{\pgfqpoint{0cm}{1.359cm}}{\pgfqpoint{0.014cm}{1.324cm}}{\pgfqpoint{0.04cm}{1.299cm}}
\pgfpathcurveto{\pgfqpoint{0.066cm}{1.273cm}}{\pgfqpoint{0.1cm}{1.258cm}}{\pgfqpoint{0.137cm}{1.258cm}}
\pgfpathcurveto{\pgfqpoint{0.173cm}{1.258cm}}{\pgfqpoint{0.207cm}{1.273cm}}{\pgfqpoint{0.233cm}{1.299cm}}
\pgfpathcurveto{\pgfqpoint{0.259cm}{1.324cm}}{\pgfqpoint{0.273cm}{1.359cm}}{\pgfqpoint{0.273cm}{1.395cm}}
\pgfusepath{fill}
\begin{pgfscope}
\pgfsetdash{}{0cm}
\pgfsetlinewidth{0.818mm}
\pgfsetmiterlimit{7.0}
\pgfpathmoveto{\pgfqpoint{0.682cm}{0.671cm}}
\pgfpathlineto{\pgfqpoint{0.679cm}{1.418cm}}
\pgfusepath{stroke}
\end{pgfscope}
\pgfpathmoveto{\pgfqpoint{0.815cm}{1.399cm}}
\pgfpathcurveto{\pgfqpoint{0.815cm}{1.435cm}}{\pgfqpoint{0.801cm}{1.47cm}}{\pgfqpoint{0.775cm}{1.496cm}}
\pgfpathcurveto{\pgfqpoint{0.75cm}{1.521cm}}{\pgfqpoint{0.715cm}{1.536cm}}{\pgfqpoint{0.679cm}{1.536cm}}
\pgfpathcurveto{\pgfqpoint{0.643cm}{1.536cm}}{\pgfqpoint{0.608cm}{1.521cm}}{\pgfqpoint{0.582cm}{1.496cm}}
\pgfpathcurveto{\pgfqpoint{0.557cm}{1.47cm}}{\pgfqpoint{0.542cm}{1.435cm}}{\pgfqpoint{0.542cm}{1.399cm}}
\pgfpathcurveto{\pgfqpoint{0.542cm}{1.363cm}}{\pgfqpoint{0.557cm}{1.328cm}}{\pgfqpoint{0.582cm}{1.302cm}}
\pgfpathcurveto{\pgfqpoint{0.608cm}{1.276cm}}{\pgfqpoint{0.643cm}{1.262cm}}{\pgfqpoint{0.679cm}{1.262cm}}
\pgfpathcurveto{\pgfqpoint{0.715cm}{1.262cm}}{\pgfqpoint{0.75cm}{1.276cm}}{\pgfqpoint{0.775cm}{1.302cm}}
\pgfpathcurveto{\pgfqpoint{0.801cm}{1.328cm}}{\pgfqpoint{0.815cm}{1.363cm}}{\pgfqpoint{0.815cm}{1.399cm}}
\pgfusepath{fill}
\pgfpathmoveto{\pgfqpoint{1.345cm}{1.371cm}}
\pgfpathcurveto{\pgfqpoint{1.345cm}{1.408cm}}{\pgfqpoint{1.331cm}{1.442cm}}{\pgfqpoint{1.305cm}{1.468cm}}
\pgfpathcurveto{\pgfqpoint{1.28cm}{1.494cm}}{\pgfqpoint{1.245cm}{1.508cm}}{\pgfqpoint{1.209cm}{1.508cm}}
\pgfpathcurveto{\pgfqpoint{1.172cm}{1.508cm}}{\pgfqpoint{1.138cm}{1.494cm}}{\pgfqpoint{1.112cm}{1.468cm}}
\pgfpathcurveto{\pgfqpoint{1.087cm}{1.442cm}}{\pgfqpoint{1.072cm}{1.408cm}}{\pgfqpoint{1.072cm}{1.371cm}}
\pgfpathcurveto{\pgfqpoint{1.072cm}{1.335cm}}{\pgfqpoint{1.087cm}{1.3cm}}{\pgfqpoint{1.112cm}{1.274cm}}
\pgfpathcurveto{\pgfqpoint{1.138cm}{1.249cm}}{\pgfqpoint{1.172cm}{1.234cm}}{\pgfqpoint{1.209cm}{1.234cm}}
\pgfpathcurveto{\pgfqpoint{1.245cm}{1.234cm}}{\pgfqpoint{1.28cm}{1.249cm}}{\pgfqpoint{1.305cm}{1.274cm}}
\pgfpathcurveto{\pgfqpoint{1.331cm}{1.3cm}}{\pgfqpoint{1.345cm}{1.335cm}}{\pgfqpoint{1.345cm}{1.371cm}}
\pgfusepath{fill}
\begin{pgfscope}
\pgfsetdash{}{0cm}
\pgfsetlinewidth{0.818mm}
\pgfsetroundcap
\pgfsetmiterlimit{4.0}
\pgfpathmoveto{\pgfqpoint{0.682cm}{0.671cm}}
\pgfpathlineto{\pgfqpoint{0.682cm}{0.042cm}}
\pgfusepath{stroke}
\end{pgfscope}
\end{pgfscope}
\end{pgfscope}
\end{pgfscope}
\end{tikzpicture}}}_{M,\varepsilon}+\zeta_{{M,\varepsilon}}$ is not suitable for the energy method. Indeed, it would introduce  the term $\llbracket X_{M, \varepsilon}^2 \rrbracket \succ
X^{\!\resizebox{0.6em}{!}{
\begin{tikzpicture}
\pgfpathmoveto{\pgfqpoint{0cm}{-0.035cm}}
\pgfpathlineto{\pgfqpoint{1.376cm}{-0.035cm}}
\pgfpathlineto{\pgfqpoint{1.376cm}{1.552cm}}
\pgfpathlineto{\pgfqpoint{0cm}{1.552cm}}
\pgfpathclose
\pgfusepath{clip}
\begin{pgfscope}
\begin{pgfscope}
\pgfpathmoveto{\pgfqpoint{0cm}{-0.035cm}}
\pgfpathlineto{\pgfqpoint{1.376cm}{-0.035cm}}
\pgfpathlineto{\pgfqpoint{1.376cm}{1.552cm}}
\pgfpathlineto{\pgfqpoint{0cm}{1.552cm}}
\pgfpathclose
\pgfusepath{clip}
\begin{pgfscope}
\begin{pgfscope}
\pgfsetdash{}{0cm}
\pgfsetlinewidth{0.818mm}
\pgfsetroundcap
\pgfsetroundjoin
\pgfsetmiterlimit{7.0}
\definecolor{eps2pgf_color}{gray}{0}\pgfsetstrokecolor{eps2pgf_color}\pgfsetfillcolor{eps2pgf_color}
\pgfpathmoveto{\pgfqpoint{0.117cm}{1.421cm}}
\pgfpathlineto{\pgfqpoint{0.682cm}{0.671cm}}
\pgfpathlineto{\pgfqpoint{1.246cm}{1.421cm}}
\pgfusepath{stroke}
\end{pgfscope}
\definecolor{eps2pgf_color}{gray}{0}\pgfsetstrokecolor{eps2pgf_color}\pgfsetfillcolor{eps2pgf_color}
\pgfpathmoveto{\pgfqpoint{0.273cm}{1.395cm}}
\pgfpathcurveto{\pgfqpoint{0.273cm}{1.432cm}}{\pgfqpoint{0.259cm}{1.467cm}}{\pgfqpoint{0.233cm}{1.492cm}}
\pgfpathcurveto{\pgfqpoint{0.207cm}{1.518cm}}{\pgfqpoint{0.173cm}{1.532cm}}{\pgfqpoint{0.137cm}{1.532cm}}
\pgfpathcurveto{\pgfqpoint{0.1cm}{1.532cm}}{\pgfqpoint{0.066cm}{1.518cm}}{\pgfqpoint{0.04cm}{1.492cm}}
\pgfpathcurveto{\pgfqpoint{0.014cm}{1.467cm}}{\pgfqpoint{0cm}{1.432cm}}{\pgfqpoint{0cm}{1.395cm}}
\pgfpathcurveto{\pgfqpoint{0cm}{1.359cm}}{\pgfqpoint{0.014cm}{1.324cm}}{\pgfqpoint{0.04cm}{1.299cm}}
\pgfpathcurveto{\pgfqpoint{0.066cm}{1.273cm}}{\pgfqpoint{0.1cm}{1.258cm}}{\pgfqpoint{0.137cm}{1.258cm}}
\pgfpathcurveto{\pgfqpoint{0.173cm}{1.258cm}}{\pgfqpoint{0.207cm}{1.273cm}}{\pgfqpoint{0.233cm}{1.299cm}}
\pgfpathcurveto{\pgfqpoint{0.259cm}{1.324cm}}{\pgfqpoint{0.273cm}{1.359cm}}{\pgfqpoint{0.273cm}{1.395cm}}
\pgfusepath{fill}
\begin{pgfscope}
\pgfsetdash{}{0cm}
\pgfsetlinewidth{0.818mm}
\pgfsetmiterlimit{7.0}
\pgfpathmoveto{\pgfqpoint{0.682cm}{0.671cm}}
\pgfpathlineto{\pgfqpoint{0.679cm}{1.418cm}}
\pgfusepath{stroke}
\end{pgfscope}
\pgfpathmoveto{\pgfqpoint{0.815cm}{1.399cm}}
\pgfpathcurveto{\pgfqpoint{0.815cm}{1.435cm}}{\pgfqpoint{0.801cm}{1.47cm}}{\pgfqpoint{0.775cm}{1.496cm}}
\pgfpathcurveto{\pgfqpoint{0.75cm}{1.521cm}}{\pgfqpoint{0.715cm}{1.536cm}}{\pgfqpoint{0.679cm}{1.536cm}}
\pgfpathcurveto{\pgfqpoint{0.643cm}{1.536cm}}{\pgfqpoint{0.608cm}{1.521cm}}{\pgfqpoint{0.582cm}{1.496cm}}
\pgfpathcurveto{\pgfqpoint{0.557cm}{1.47cm}}{\pgfqpoint{0.542cm}{1.435cm}}{\pgfqpoint{0.542cm}{1.399cm}}
\pgfpathcurveto{\pgfqpoint{0.542cm}{1.363cm}}{\pgfqpoint{0.557cm}{1.328cm}}{\pgfqpoint{0.582cm}{1.302cm}}
\pgfpathcurveto{\pgfqpoint{0.608cm}{1.276cm}}{\pgfqpoint{0.643cm}{1.262cm}}{\pgfqpoint{0.679cm}{1.262cm}}
\pgfpathcurveto{\pgfqpoint{0.715cm}{1.262cm}}{\pgfqpoint{0.75cm}{1.276cm}}{\pgfqpoint{0.775cm}{1.302cm}}
\pgfpathcurveto{\pgfqpoint{0.801cm}{1.328cm}}{\pgfqpoint{0.815cm}{1.363cm}}{\pgfqpoint{0.815cm}{1.399cm}}
\pgfusepath{fill}
\pgfpathmoveto{\pgfqpoint{1.345cm}{1.371cm}}
\pgfpathcurveto{\pgfqpoint{1.345cm}{1.408cm}}{\pgfqpoint{1.331cm}{1.442cm}}{\pgfqpoint{1.305cm}{1.468cm}}
\pgfpathcurveto{\pgfqpoint{1.28cm}{1.494cm}}{\pgfqpoint{1.245cm}{1.508cm}}{\pgfqpoint{1.209cm}{1.508cm}}
\pgfpathcurveto{\pgfqpoint{1.172cm}{1.508cm}}{\pgfqpoint{1.138cm}{1.494cm}}{\pgfqpoint{1.112cm}{1.468cm}}
\pgfpathcurveto{\pgfqpoint{1.087cm}{1.442cm}}{\pgfqpoint{1.072cm}{1.408cm}}{\pgfqpoint{1.072cm}{1.371cm}}
\pgfpathcurveto{\pgfqpoint{1.072cm}{1.335cm}}{\pgfqpoint{1.087cm}{1.3cm}}{\pgfqpoint{1.112cm}{1.274cm}}
\pgfpathcurveto{\pgfqpoint{1.138cm}{1.249cm}}{\pgfqpoint{1.172cm}{1.234cm}}{\pgfqpoint{1.209cm}{1.234cm}}
\pgfpathcurveto{\pgfqpoint{1.245cm}{1.234cm}}{\pgfqpoint{1.28cm}{1.249cm}}{\pgfqpoint{1.305cm}{1.274cm}}
\pgfpathcurveto{\pgfqpoint{1.331cm}{1.3cm}}{\pgfqpoint{1.345cm}{1.335cm}}{\pgfqpoint{1.345cm}{1.371cm}}
\pgfusepath{fill}
\begin{pgfscope}
\pgfsetdash{}{0cm}
\pgfsetlinewidth{0.818mm}
\pgfsetroundcap
\pgfsetmiterlimit{4.0}
\pgfpathmoveto{\pgfqpoint{0.682cm}{0.671cm}}
\pgfpathlineto{\pgfqpoint{0.682cm}{0.042cm}}
\pgfusepath{stroke}
\end{pgfscope}
\end{pgfscope}
\end{pgfscope}
\end{pgfscope}
\end{tikzpicture}}}_{M, \varepsilon}$ which cannot be cancelled or controlled by the available quantities. We overcome this issue by working rather with the decomposition $\varphi_{M,\varepsilon}=X_{M,\varepsilon}+Y_{M,\varepsilon}+\phi_{M,\varepsilon}$ defined in the sequel. Note that a similar modification of the paracontrolled ansatz has been necessary to construct a renormalized control problem for the KPZ equation in \cite{gubinelli_kpz_2017}. Here, the price to pay is that the auxiliary variables $Y_{M,\varepsilon}$, $\phi_{M,\varepsilon}$ are not stationary. Thus, in Section~\ref{ss:tight} we go back to the stationary decomposition $\varphi_{M,\varepsilon}=X_{M,\varepsilon}-\lambda X^{\!\resizebox{0.6em}{!}{
\begin{tikzpicture}
\pgfpathmoveto{\pgfqpoint{0cm}{-0.035cm}}
\pgfpathlineto{\pgfqpoint{1.376cm}{-0.035cm}}
\pgfpathlineto{\pgfqpoint{1.376cm}{1.552cm}}
\pgfpathlineto{\pgfqpoint{0cm}{1.552cm}}
\pgfpathclose
\pgfusepath{clip}
\begin{pgfscope}
\begin{pgfscope}
\pgfpathmoveto{\pgfqpoint{0cm}{-0.035cm}}
\pgfpathlineto{\pgfqpoint{1.376cm}{-0.035cm}}
\pgfpathlineto{\pgfqpoint{1.376cm}{1.552cm}}
\pgfpathlineto{\pgfqpoint{0cm}{1.552cm}}
\pgfpathclose
\pgfusepath{clip}
\begin{pgfscope}
\begin{pgfscope}
\pgfsetdash{}{0cm}
\pgfsetlinewidth{0.818mm}
\pgfsetroundcap
\pgfsetroundjoin
\pgfsetmiterlimit{7.0}
\definecolor{eps2pgf_color}{gray}{0}\pgfsetstrokecolor{eps2pgf_color}\pgfsetfillcolor{eps2pgf_color}
\pgfpathmoveto{\pgfqpoint{0.117cm}{1.421cm}}
\pgfpathlineto{\pgfqpoint{0.682cm}{0.671cm}}
\pgfpathlineto{\pgfqpoint{1.246cm}{1.421cm}}
\pgfusepath{stroke}
\end{pgfscope}
\definecolor{eps2pgf_color}{gray}{0}\pgfsetstrokecolor{eps2pgf_color}\pgfsetfillcolor{eps2pgf_color}
\pgfpathmoveto{\pgfqpoint{0.273cm}{1.395cm}}
\pgfpathcurveto{\pgfqpoint{0.273cm}{1.432cm}}{\pgfqpoint{0.259cm}{1.467cm}}{\pgfqpoint{0.233cm}{1.492cm}}
\pgfpathcurveto{\pgfqpoint{0.207cm}{1.518cm}}{\pgfqpoint{0.173cm}{1.532cm}}{\pgfqpoint{0.137cm}{1.532cm}}
\pgfpathcurveto{\pgfqpoint{0.1cm}{1.532cm}}{\pgfqpoint{0.066cm}{1.518cm}}{\pgfqpoint{0.04cm}{1.492cm}}
\pgfpathcurveto{\pgfqpoint{0.014cm}{1.467cm}}{\pgfqpoint{0cm}{1.432cm}}{\pgfqpoint{0cm}{1.395cm}}
\pgfpathcurveto{\pgfqpoint{0cm}{1.359cm}}{\pgfqpoint{0.014cm}{1.324cm}}{\pgfqpoint{0.04cm}{1.299cm}}
\pgfpathcurveto{\pgfqpoint{0.066cm}{1.273cm}}{\pgfqpoint{0.1cm}{1.258cm}}{\pgfqpoint{0.137cm}{1.258cm}}
\pgfpathcurveto{\pgfqpoint{0.173cm}{1.258cm}}{\pgfqpoint{0.207cm}{1.273cm}}{\pgfqpoint{0.233cm}{1.299cm}}
\pgfpathcurveto{\pgfqpoint{0.259cm}{1.324cm}}{\pgfqpoint{0.273cm}{1.359cm}}{\pgfqpoint{0.273cm}{1.395cm}}
\pgfusepath{fill}
\begin{pgfscope}
\pgfsetdash{}{0cm}
\pgfsetlinewidth{0.818mm}
\pgfsetmiterlimit{7.0}
\pgfpathmoveto{\pgfqpoint{0.682cm}{0.671cm}}
\pgfpathlineto{\pgfqpoint{0.679cm}{1.418cm}}
\pgfusepath{stroke}
\end{pgfscope}
\pgfpathmoveto{\pgfqpoint{0.815cm}{1.399cm}}
\pgfpathcurveto{\pgfqpoint{0.815cm}{1.435cm}}{\pgfqpoint{0.801cm}{1.47cm}}{\pgfqpoint{0.775cm}{1.496cm}}
\pgfpathcurveto{\pgfqpoint{0.75cm}{1.521cm}}{\pgfqpoint{0.715cm}{1.536cm}}{\pgfqpoint{0.679cm}{1.536cm}}
\pgfpathcurveto{\pgfqpoint{0.643cm}{1.536cm}}{\pgfqpoint{0.608cm}{1.521cm}}{\pgfqpoint{0.582cm}{1.496cm}}
\pgfpathcurveto{\pgfqpoint{0.557cm}{1.47cm}}{\pgfqpoint{0.542cm}{1.435cm}}{\pgfqpoint{0.542cm}{1.399cm}}
\pgfpathcurveto{\pgfqpoint{0.542cm}{1.363cm}}{\pgfqpoint{0.557cm}{1.328cm}}{\pgfqpoint{0.582cm}{1.302cm}}
\pgfpathcurveto{\pgfqpoint{0.608cm}{1.276cm}}{\pgfqpoint{0.643cm}{1.262cm}}{\pgfqpoint{0.679cm}{1.262cm}}
\pgfpathcurveto{\pgfqpoint{0.715cm}{1.262cm}}{\pgfqpoint{0.75cm}{1.276cm}}{\pgfqpoint{0.775cm}{1.302cm}}
\pgfpathcurveto{\pgfqpoint{0.801cm}{1.328cm}}{\pgfqpoint{0.815cm}{1.363cm}}{\pgfqpoint{0.815cm}{1.399cm}}
\pgfusepath{fill}
\pgfpathmoveto{\pgfqpoint{1.345cm}{1.371cm}}
\pgfpathcurveto{\pgfqpoint{1.345cm}{1.408cm}}{\pgfqpoint{1.331cm}{1.442cm}}{\pgfqpoint{1.305cm}{1.468cm}}
\pgfpathcurveto{\pgfqpoint{1.28cm}{1.494cm}}{\pgfqpoint{1.245cm}{1.508cm}}{\pgfqpoint{1.209cm}{1.508cm}}
\pgfpathcurveto{\pgfqpoint{1.172cm}{1.508cm}}{\pgfqpoint{1.138cm}{1.494cm}}{\pgfqpoint{1.112cm}{1.468cm}}
\pgfpathcurveto{\pgfqpoint{1.087cm}{1.442cm}}{\pgfqpoint{1.072cm}{1.408cm}}{\pgfqpoint{1.072cm}{1.371cm}}
\pgfpathcurveto{\pgfqpoint{1.072cm}{1.335cm}}{\pgfqpoint{1.087cm}{1.3cm}}{\pgfqpoint{1.112cm}{1.274cm}}
\pgfpathcurveto{\pgfqpoint{1.138cm}{1.249cm}}{\pgfqpoint{1.172cm}{1.234cm}}{\pgfqpoint{1.209cm}{1.234cm}}
\pgfpathcurveto{\pgfqpoint{1.245cm}{1.234cm}}{\pgfqpoint{1.28cm}{1.249cm}}{\pgfqpoint{1.305cm}{1.274cm}}
\pgfpathcurveto{\pgfqpoint{1.331cm}{1.3cm}}{\pgfqpoint{1.345cm}{1.335cm}}{\pgfqpoint{1.345cm}{1.371cm}}
\pgfusepath{fill}
\begin{pgfscope}
\pgfsetdash{}{0cm}
\pgfsetlinewidth{0.818mm}
\pgfsetroundcap
\pgfsetmiterlimit{4.0}
\pgfpathmoveto{\pgfqpoint{0.682cm}{0.671cm}}
\pgfpathlineto{\pgfqpoint{0.682cm}{0.042cm}}
\pgfusepath{stroke}
\end{pgfscope}
\end{pgfscope}
\end{pgfscope}
\end{pgfscope}
\end{tikzpicture}}}_{M,\varepsilon}+\zeta_{{M,\varepsilon}}$. 

If $\UU^{\varepsilon}_{>}$ is a localizer defined for some given constant $L >
0$ according to Lemma~\ref{lem:loc}, we let $Y_{M, \varepsilon}$ be the
solution of \eqref{eq:Y1} 
hence
\begin{equation}\label{eq:YY}
Y_{M, \varepsilon} = -\lambda
X_{M, \varepsilon}^{\!\resizebox{0.6em}{!}{
\begin{tikzpicture}
\pgfpathmoveto{\pgfqpoint{0cm}{-0.035cm}}
\pgfpathlineto{\pgfqpoint{1.376cm}{-0.035cm}}
\pgfpathlineto{\pgfqpoint{1.376cm}{1.552cm}}
\pgfpathlineto{\pgfqpoint{0cm}{1.552cm}}
\pgfpathclose
\pgfusepath{clip}
\begin{pgfscope}
\begin{pgfscope}
\pgfpathmoveto{\pgfqpoint{0cm}{-0.035cm}}
\pgfpathlineto{\pgfqpoint{1.376cm}{-0.035cm}}
\pgfpathlineto{\pgfqpoint{1.376cm}{1.552cm}}
\pgfpathlineto{\pgfqpoint{0cm}{1.552cm}}
\pgfpathclose
\pgfusepath{clip}
\begin{pgfscope}
\begin{pgfscope}
\pgfsetdash{}{0cm}
\pgfsetlinewidth{0.818mm}
\pgfsetroundcap
\pgfsetroundjoin
\pgfsetmiterlimit{7.0}
\definecolor{eps2pgf_color}{gray}{0}\pgfsetstrokecolor{eps2pgf_color}\pgfsetfillcolor{eps2pgf_color}
\pgfpathmoveto{\pgfqpoint{0.117cm}{1.421cm}}
\pgfpathlineto{\pgfqpoint{0.682cm}{0.671cm}}
\pgfpathlineto{\pgfqpoint{1.246cm}{1.421cm}}
\pgfusepath{stroke}
\end{pgfscope}
\definecolor{eps2pgf_color}{gray}{0}\pgfsetstrokecolor{eps2pgf_color}\pgfsetfillcolor{eps2pgf_color}
\pgfpathmoveto{\pgfqpoint{0.273cm}{1.395cm}}
\pgfpathcurveto{\pgfqpoint{0.273cm}{1.432cm}}{\pgfqpoint{0.259cm}{1.467cm}}{\pgfqpoint{0.233cm}{1.492cm}}
\pgfpathcurveto{\pgfqpoint{0.207cm}{1.518cm}}{\pgfqpoint{0.173cm}{1.532cm}}{\pgfqpoint{0.137cm}{1.532cm}}
\pgfpathcurveto{\pgfqpoint{0.1cm}{1.532cm}}{\pgfqpoint{0.066cm}{1.518cm}}{\pgfqpoint{0.04cm}{1.492cm}}
\pgfpathcurveto{\pgfqpoint{0.014cm}{1.467cm}}{\pgfqpoint{0cm}{1.432cm}}{\pgfqpoint{0cm}{1.395cm}}
\pgfpathcurveto{\pgfqpoint{0cm}{1.359cm}}{\pgfqpoint{0.014cm}{1.324cm}}{\pgfqpoint{0.04cm}{1.299cm}}
\pgfpathcurveto{\pgfqpoint{0.066cm}{1.273cm}}{\pgfqpoint{0.1cm}{1.258cm}}{\pgfqpoint{0.137cm}{1.258cm}}
\pgfpathcurveto{\pgfqpoint{0.173cm}{1.258cm}}{\pgfqpoint{0.207cm}{1.273cm}}{\pgfqpoint{0.233cm}{1.299cm}}
\pgfpathcurveto{\pgfqpoint{0.259cm}{1.324cm}}{\pgfqpoint{0.273cm}{1.359cm}}{\pgfqpoint{0.273cm}{1.395cm}}
\pgfusepath{fill}
\begin{pgfscope}
\pgfsetdash{}{0cm}
\pgfsetlinewidth{0.818mm}
\pgfsetmiterlimit{7.0}
\pgfpathmoveto{\pgfqpoint{0.682cm}{0.671cm}}
\pgfpathlineto{\pgfqpoint{0.679cm}{1.418cm}}
\pgfusepath{stroke}
\end{pgfscope}
\pgfpathmoveto{\pgfqpoint{0.815cm}{1.399cm}}
\pgfpathcurveto{\pgfqpoint{0.815cm}{1.435cm}}{\pgfqpoint{0.801cm}{1.47cm}}{\pgfqpoint{0.775cm}{1.496cm}}
\pgfpathcurveto{\pgfqpoint{0.75cm}{1.521cm}}{\pgfqpoint{0.715cm}{1.536cm}}{\pgfqpoint{0.679cm}{1.536cm}}
\pgfpathcurveto{\pgfqpoint{0.643cm}{1.536cm}}{\pgfqpoint{0.608cm}{1.521cm}}{\pgfqpoint{0.582cm}{1.496cm}}
\pgfpathcurveto{\pgfqpoint{0.557cm}{1.47cm}}{\pgfqpoint{0.542cm}{1.435cm}}{\pgfqpoint{0.542cm}{1.399cm}}
\pgfpathcurveto{\pgfqpoint{0.542cm}{1.363cm}}{\pgfqpoint{0.557cm}{1.328cm}}{\pgfqpoint{0.582cm}{1.302cm}}
\pgfpathcurveto{\pgfqpoint{0.608cm}{1.276cm}}{\pgfqpoint{0.643cm}{1.262cm}}{\pgfqpoint{0.679cm}{1.262cm}}
\pgfpathcurveto{\pgfqpoint{0.715cm}{1.262cm}}{\pgfqpoint{0.75cm}{1.276cm}}{\pgfqpoint{0.775cm}{1.302cm}}
\pgfpathcurveto{\pgfqpoint{0.801cm}{1.328cm}}{\pgfqpoint{0.815cm}{1.363cm}}{\pgfqpoint{0.815cm}{1.399cm}}
\pgfusepath{fill}
\pgfpathmoveto{\pgfqpoint{1.345cm}{1.371cm}}
\pgfpathcurveto{\pgfqpoint{1.345cm}{1.408cm}}{\pgfqpoint{1.331cm}{1.442cm}}{\pgfqpoint{1.305cm}{1.468cm}}
\pgfpathcurveto{\pgfqpoint{1.28cm}{1.494cm}}{\pgfqpoint{1.245cm}{1.508cm}}{\pgfqpoint{1.209cm}{1.508cm}}
\pgfpathcurveto{\pgfqpoint{1.172cm}{1.508cm}}{\pgfqpoint{1.138cm}{1.494cm}}{\pgfqpoint{1.112cm}{1.468cm}}
\pgfpathcurveto{\pgfqpoint{1.087cm}{1.442cm}}{\pgfqpoint{1.072cm}{1.408cm}}{\pgfqpoint{1.072cm}{1.371cm}}
\pgfpathcurveto{\pgfqpoint{1.072cm}{1.335cm}}{\pgfqpoint{1.087cm}{1.3cm}}{\pgfqpoint{1.112cm}{1.274cm}}
\pgfpathcurveto{\pgfqpoint{1.138cm}{1.249cm}}{\pgfqpoint{1.172cm}{1.234cm}}{\pgfqpoint{1.209cm}{1.234cm}}
\pgfpathcurveto{\pgfqpoint{1.245cm}{1.234cm}}{\pgfqpoint{1.28cm}{1.249cm}}{\pgfqpoint{1.305cm}{1.274cm}}
\pgfpathcurveto{\pgfqpoint{1.331cm}{1.3cm}}{\pgfqpoint{1.345cm}{1.335cm}}{\pgfqpoint{1.345cm}{1.371cm}}
\pgfusepath{fill}
\begin{pgfscope}
\pgfsetdash{}{0cm}
\pgfsetlinewidth{0.818mm}
\pgfsetroundcap
\pgfsetmiterlimit{4.0}
\pgfpathmoveto{\pgfqpoint{0.682cm}{0.671cm}}
\pgfpathlineto{\pgfqpoint{0.682cm}{0.042cm}}
\pgfusepath{stroke}
\end{pgfscope}
\end{pgfscope}
\end{pgfscope}
\end{pgfscope}
\end{tikzpicture}}} - \LL_{\varepsilon}^{- 1} [ 3\lambda (
\UU^{\varepsilon}_{>} \llbracket X_{M, \varepsilon}^2 \rrbracket ) \succ Y_{M,
\varepsilon} ].
\end{equation}
Note that this is an equation for $Y_{M, \varepsilon}$,
which also implies that $Y_{M, \varepsilon}$ is not a polynomial of the
Gaussian noise. However, as shown in the following lemma, $Y_{M, \varepsilon}$
can be constructed as a fixed point provided $L$ is large enough.

\begin{lemma}
  \label{lem:Y1}
  There exists $L_{0}=L_{0}(\lambda)\geqslant 0$ and $L=L(\lambda,M,\varepsilon) \geqslant 0$ with a (not relabeled) subsequence satisfying $L(\lambda,M,\varepsilon)\to L_{0}$ as $\varepsilon\to0$, $M\to\infty$, such that {\eqref{eq:Y1}} with $\UU^{\varepsilon}_{>}$ determined by  $L$ has a unique
  solution $Y_{M, \varepsilon}$ that belongs to $C_T \CC^{1 / 2 - \kappa}
  (\rho^{\sigma}) \cap C_T^{\beta / 2} L^{\infty} (\rho^{\sigma})$.
  Furthermore,
  \[ \| Y_{M, \varepsilon} \|_{C_T \CC^{1 / 2 - \kappa, \varepsilon}
     (\rho^{\sigma})} \lesssim \lambda \| X_{M, \varepsilon}^{\!\resizebox{0.6em}{!}{
\begin{tikzpicture}
\pgfpathmoveto{\pgfqpoint{0cm}{-0.035cm}}
\pgfpathlineto{\pgfqpoint{1.376cm}{-0.035cm}}
\pgfpathlineto{\pgfqpoint{1.376cm}{1.552cm}}
\pgfpathlineto{\pgfqpoint{0cm}{1.552cm}}
\pgfpathclose
\pgfusepath{clip}
\begin{pgfscope}
\begin{pgfscope}
\pgfpathmoveto{\pgfqpoint{0cm}{-0.035cm}}
\pgfpathlineto{\pgfqpoint{1.376cm}{-0.035cm}}
\pgfpathlineto{\pgfqpoint{1.376cm}{1.552cm}}
\pgfpathlineto{\pgfqpoint{0cm}{1.552cm}}
\pgfpathclose
\pgfusepath{clip}
\begin{pgfscope}
\begin{pgfscope}
\pgfsetdash{}{0cm}
\pgfsetlinewidth{0.818mm}
\pgfsetroundcap
\pgfsetroundjoin
\pgfsetmiterlimit{7.0}
\definecolor{eps2pgf_color}{gray}{0}\pgfsetstrokecolor{eps2pgf_color}\pgfsetfillcolor{eps2pgf_color}
\pgfpathmoveto{\pgfqpoint{0.117cm}{1.421cm}}
\pgfpathlineto{\pgfqpoint{0.682cm}{0.671cm}}
\pgfpathlineto{\pgfqpoint{1.246cm}{1.421cm}}
\pgfusepath{stroke}
\end{pgfscope}
\definecolor{eps2pgf_color}{gray}{0}\pgfsetstrokecolor{eps2pgf_color}\pgfsetfillcolor{eps2pgf_color}
\pgfpathmoveto{\pgfqpoint{0.273cm}{1.395cm}}
\pgfpathcurveto{\pgfqpoint{0.273cm}{1.432cm}}{\pgfqpoint{0.259cm}{1.467cm}}{\pgfqpoint{0.233cm}{1.492cm}}
\pgfpathcurveto{\pgfqpoint{0.207cm}{1.518cm}}{\pgfqpoint{0.173cm}{1.532cm}}{\pgfqpoint{0.137cm}{1.532cm}}
\pgfpathcurveto{\pgfqpoint{0.1cm}{1.532cm}}{\pgfqpoint{0.066cm}{1.518cm}}{\pgfqpoint{0.04cm}{1.492cm}}
\pgfpathcurveto{\pgfqpoint{0.014cm}{1.467cm}}{\pgfqpoint{0cm}{1.432cm}}{\pgfqpoint{0cm}{1.395cm}}
\pgfpathcurveto{\pgfqpoint{0cm}{1.359cm}}{\pgfqpoint{0.014cm}{1.324cm}}{\pgfqpoint{0.04cm}{1.299cm}}
\pgfpathcurveto{\pgfqpoint{0.066cm}{1.273cm}}{\pgfqpoint{0.1cm}{1.258cm}}{\pgfqpoint{0.137cm}{1.258cm}}
\pgfpathcurveto{\pgfqpoint{0.173cm}{1.258cm}}{\pgfqpoint{0.207cm}{1.273cm}}{\pgfqpoint{0.233cm}{1.299cm}}
\pgfpathcurveto{\pgfqpoint{0.259cm}{1.324cm}}{\pgfqpoint{0.273cm}{1.359cm}}{\pgfqpoint{0.273cm}{1.395cm}}
\pgfusepath{fill}
\begin{pgfscope}
\pgfsetdash{}{0cm}
\pgfsetlinewidth{0.818mm}
\pgfsetmiterlimit{7.0}
\pgfpathmoveto{\pgfqpoint{0.682cm}{0.671cm}}
\pgfpathlineto{\pgfqpoint{0.679cm}{1.418cm}}
\pgfusepath{stroke}
\end{pgfscope}
\pgfpathmoveto{\pgfqpoint{0.815cm}{1.399cm}}
\pgfpathcurveto{\pgfqpoint{0.815cm}{1.435cm}}{\pgfqpoint{0.801cm}{1.47cm}}{\pgfqpoint{0.775cm}{1.496cm}}
\pgfpathcurveto{\pgfqpoint{0.75cm}{1.521cm}}{\pgfqpoint{0.715cm}{1.536cm}}{\pgfqpoint{0.679cm}{1.536cm}}
\pgfpathcurveto{\pgfqpoint{0.643cm}{1.536cm}}{\pgfqpoint{0.608cm}{1.521cm}}{\pgfqpoint{0.582cm}{1.496cm}}
\pgfpathcurveto{\pgfqpoint{0.557cm}{1.47cm}}{\pgfqpoint{0.542cm}{1.435cm}}{\pgfqpoint{0.542cm}{1.399cm}}
\pgfpathcurveto{\pgfqpoint{0.542cm}{1.363cm}}{\pgfqpoint{0.557cm}{1.328cm}}{\pgfqpoint{0.582cm}{1.302cm}}
\pgfpathcurveto{\pgfqpoint{0.608cm}{1.276cm}}{\pgfqpoint{0.643cm}{1.262cm}}{\pgfqpoint{0.679cm}{1.262cm}}
\pgfpathcurveto{\pgfqpoint{0.715cm}{1.262cm}}{\pgfqpoint{0.75cm}{1.276cm}}{\pgfqpoint{0.775cm}{1.302cm}}
\pgfpathcurveto{\pgfqpoint{0.801cm}{1.328cm}}{\pgfqpoint{0.815cm}{1.363cm}}{\pgfqpoint{0.815cm}{1.399cm}}
\pgfusepath{fill}
\pgfpathmoveto{\pgfqpoint{1.345cm}{1.371cm}}
\pgfpathcurveto{\pgfqpoint{1.345cm}{1.408cm}}{\pgfqpoint{1.331cm}{1.442cm}}{\pgfqpoint{1.305cm}{1.468cm}}
\pgfpathcurveto{\pgfqpoint{1.28cm}{1.494cm}}{\pgfqpoint{1.245cm}{1.508cm}}{\pgfqpoint{1.209cm}{1.508cm}}
\pgfpathcurveto{\pgfqpoint{1.172cm}{1.508cm}}{\pgfqpoint{1.138cm}{1.494cm}}{\pgfqpoint{1.112cm}{1.468cm}}
\pgfpathcurveto{\pgfqpoint{1.087cm}{1.442cm}}{\pgfqpoint{1.072cm}{1.408cm}}{\pgfqpoint{1.072cm}{1.371cm}}
\pgfpathcurveto{\pgfqpoint{1.072cm}{1.335cm}}{\pgfqpoint{1.087cm}{1.3cm}}{\pgfqpoint{1.112cm}{1.274cm}}
\pgfpathcurveto{\pgfqpoint{1.138cm}{1.249cm}}{\pgfqpoint{1.172cm}{1.234cm}}{\pgfqpoint{1.209cm}{1.234cm}}
\pgfpathcurveto{\pgfqpoint{1.245cm}{1.234cm}}{\pgfqpoint{1.28cm}{1.249cm}}{\pgfqpoint{1.305cm}{1.274cm}}
\pgfpathcurveto{\pgfqpoint{1.331cm}{1.3cm}}{\pgfqpoint{1.345cm}{1.335cm}}{\pgfqpoint{1.345cm}{1.371cm}}
\pgfusepath{fill}
\begin{pgfscope}
\pgfsetdash{}{0cm}
\pgfsetlinewidth{0.818mm}
\pgfsetroundcap
\pgfsetmiterlimit{4.0}
\pgfpathmoveto{\pgfqpoint{0.682cm}{0.671cm}}
\pgfpathlineto{\pgfqpoint{0.682cm}{0.042cm}}
\pgfusepath{stroke}
\end{pgfscope}
\end{pgfscope}
\end{pgfscope}
\end{pgfscope}
\end{tikzpicture}}}
     \|_{C_T \CC^{1 / 2 - \kappa, \varepsilon} (\rho^{\sigma})}, \]
  \[ \| Y_{M, \varepsilon} \|_{C_T^{\beta / 2} L^{\infty, \varepsilon}
     (\rho^{\sigma})} \lesssim \lambda[ \| X_{M, \varepsilon}^{\!\resizebox{0.6em}{!}{
\begin{tikzpicture}
\pgfpathmoveto{\pgfqpoint{0cm}{-0.035cm}}
\pgfpathlineto{\pgfqpoint{1.376cm}{-0.035cm}}
\pgfpathlineto{\pgfqpoint{1.376cm}{1.552cm}}
\pgfpathlineto{\pgfqpoint{0cm}{1.552cm}}
\pgfpathclose
\pgfusepath{clip}
\begin{pgfscope}
\begin{pgfscope}
\pgfpathmoveto{\pgfqpoint{0cm}{-0.035cm}}
\pgfpathlineto{\pgfqpoint{1.376cm}{-0.035cm}}
\pgfpathlineto{\pgfqpoint{1.376cm}{1.552cm}}
\pgfpathlineto{\pgfqpoint{0cm}{1.552cm}}
\pgfpathclose
\pgfusepath{clip}
\begin{pgfscope}
\begin{pgfscope}
\pgfsetdash{}{0cm}
\pgfsetlinewidth{0.818mm}
\pgfsetroundcap
\pgfsetroundjoin
\pgfsetmiterlimit{7.0}
\definecolor{eps2pgf_color}{gray}{0}\pgfsetstrokecolor{eps2pgf_color}\pgfsetfillcolor{eps2pgf_color}
\pgfpathmoveto{\pgfqpoint{0.117cm}{1.421cm}}
\pgfpathlineto{\pgfqpoint{0.682cm}{0.671cm}}
\pgfpathlineto{\pgfqpoint{1.246cm}{1.421cm}}
\pgfusepath{stroke}
\end{pgfscope}
\definecolor{eps2pgf_color}{gray}{0}\pgfsetstrokecolor{eps2pgf_color}\pgfsetfillcolor{eps2pgf_color}
\pgfpathmoveto{\pgfqpoint{0.273cm}{1.395cm}}
\pgfpathcurveto{\pgfqpoint{0.273cm}{1.432cm}}{\pgfqpoint{0.259cm}{1.467cm}}{\pgfqpoint{0.233cm}{1.492cm}}
\pgfpathcurveto{\pgfqpoint{0.207cm}{1.518cm}}{\pgfqpoint{0.173cm}{1.532cm}}{\pgfqpoint{0.137cm}{1.532cm}}
\pgfpathcurveto{\pgfqpoint{0.1cm}{1.532cm}}{\pgfqpoint{0.066cm}{1.518cm}}{\pgfqpoint{0.04cm}{1.492cm}}
\pgfpathcurveto{\pgfqpoint{0.014cm}{1.467cm}}{\pgfqpoint{0cm}{1.432cm}}{\pgfqpoint{0cm}{1.395cm}}
\pgfpathcurveto{\pgfqpoint{0cm}{1.359cm}}{\pgfqpoint{0.014cm}{1.324cm}}{\pgfqpoint{0.04cm}{1.299cm}}
\pgfpathcurveto{\pgfqpoint{0.066cm}{1.273cm}}{\pgfqpoint{0.1cm}{1.258cm}}{\pgfqpoint{0.137cm}{1.258cm}}
\pgfpathcurveto{\pgfqpoint{0.173cm}{1.258cm}}{\pgfqpoint{0.207cm}{1.273cm}}{\pgfqpoint{0.233cm}{1.299cm}}
\pgfpathcurveto{\pgfqpoint{0.259cm}{1.324cm}}{\pgfqpoint{0.273cm}{1.359cm}}{\pgfqpoint{0.273cm}{1.395cm}}
\pgfusepath{fill}
\begin{pgfscope}
\pgfsetdash{}{0cm}
\pgfsetlinewidth{0.818mm}
\pgfsetmiterlimit{7.0}
\pgfpathmoveto{\pgfqpoint{0.682cm}{0.671cm}}
\pgfpathlineto{\pgfqpoint{0.679cm}{1.418cm}}
\pgfusepath{stroke}
\end{pgfscope}
\pgfpathmoveto{\pgfqpoint{0.815cm}{1.399cm}}
\pgfpathcurveto{\pgfqpoint{0.815cm}{1.435cm}}{\pgfqpoint{0.801cm}{1.47cm}}{\pgfqpoint{0.775cm}{1.496cm}}
\pgfpathcurveto{\pgfqpoint{0.75cm}{1.521cm}}{\pgfqpoint{0.715cm}{1.536cm}}{\pgfqpoint{0.679cm}{1.536cm}}
\pgfpathcurveto{\pgfqpoint{0.643cm}{1.536cm}}{\pgfqpoint{0.608cm}{1.521cm}}{\pgfqpoint{0.582cm}{1.496cm}}
\pgfpathcurveto{\pgfqpoint{0.557cm}{1.47cm}}{\pgfqpoint{0.542cm}{1.435cm}}{\pgfqpoint{0.542cm}{1.399cm}}
\pgfpathcurveto{\pgfqpoint{0.542cm}{1.363cm}}{\pgfqpoint{0.557cm}{1.328cm}}{\pgfqpoint{0.582cm}{1.302cm}}
\pgfpathcurveto{\pgfqpoint{0.608cm}{1.276cm}}{\pgfqpoint{0.643cm}{1.262cm}}{\pgfqpoint{0.679cm}{1.262cm}}
\pgfpathcurveto{\pgfqpoint{0.715cm}{1.262cm}}{\pgfqpoint{0.75cm}{1.276cm}}{\pgfqpoint{0.775cm}{1.302cm}}
\pgfpathcurveto{\pgfqpoint{0.801cm}{1.328cm}}{\pgfqpoint{0.815cm}{1.363cm}}{\pgfqpoint{0.815cm}{1.399cm}}
\pgfusepath{fill}
\pgfpathmoveto{\pgfqpoint{1.345cm}{1.371cm}}
\pgfpathcurveto{\pgfqpoint{1.345cm}{1.408cm}}{\pgfqpoint{1.331cm}{1.442cm}}{\pgfqpoint{1.305cm}{1.468cm}}
\pgfpathcurveto{\pgfqpoint{1.28cm}{1.494cm}}{\pgfqpoint{1.245cm}{1.508cm}}{\pgfqpoint{1.209cm}{1.508cm}}
\pgfpathcurveto{\pgfqpoint{1.172cm}{1.508cm}}{\pgfqpoint{1.138cm}{1.494cm}}{\pgfqpoint{1.112cm}{1.468cm}}
\pgfpathcurveto{\pgfqpoint{1.087cm}{1.442cm}}{\pgfqpoint{1.072cm}{1.408cm}}{\pgfqpoint{1.072cm}{1.371cm}}
\pgfpathcurveto{\pgfqpoint{1.072cm}{1.335cm}}{\pgfqpoint{1.087cm}{1.3cm}}{\pgfqpoint{1.112cm}{1.274cm}}
\pgfpathcurveto{\pgfqpoint{1.138cm}{1.249cm}}{\pgfqpoint{1.172cm}{1.234cm}}{\pgfqpoint{1.209cm}{1.234cm}}
\pgfpathcurveto{\pgfqpoint{1.245cm}{1.234cm}}{\pgfqpoint{1.28cm}{1.249cm}}{\pgfqpoint{1.305cm}{1.274cm}}
\pgfpathcurveto{\pgfqpoint{1.331cm}{1.3cm}}{\pgfqpoint{1.345cm}{1.335cm}}{\pgfqpoint{1.345cm}{1.371cm}}
\pgfusepath{fill}
\begin{pgfscope}
\pgfsetdash{}{0cm}
\pgfsetlinewidth{0.818mm}
\pgfsetroundcap
\pgfsetmiterlimit{4.0}
\pgfpathmoveto{\pgfqpoint{0.682cm}{0.671cm}}
\pgfpathlineto{\pgfqpoint{0.682cm}{0.042cm}}
\pgfusepath{stroke}
\end{pgfscope}
\end{pgfscope}
\end{pgfscope}
\end{pgfscope}
\end{tikzpicture}}}
     \|_{C_T \CC^{1 / 2 - \kappa, \varepsilon} (\rho^{\sigma})} +
     \| X_{M, \varepsilon}^{\!\resizebox{0.6em}{!}{
\begin{tikzpicture}
\pgfpathmoveto{\pgfqpoint{0cm}{-0.035cm}}
\pgfpathlineto{\pgfqpoint{1.376cm}{-0.035cm}}
\pgfpathlineto{\pgfqpoint{1.376cm}{1.552cm}}
\pgfpathlineto{\pgfqpoint{0cm}{1.552cm}}
\pgfpathclose
\pgfusepath{clip}
\begin{pgfscope}
\begin{pgfscope}
\pgfpathmoveto{\pgfqpoint{0cm}{-0.035cm}}
\pgfpathlineto{\pgfqpoint{1.376cm}{-0.035cm}}
\pgfpathlineto{\pgfqpoint{1.376cm}{1.552cm}}
\pgfpathlineto{\pgfqpoint{0cm}{1.552cm}}
\pgfpathclose
\pgfusepath{clip}
\begin{pgfscope}
\begin{pgfscope}
\pgfsetdash{}{0cm}
\pgfsetlinewidth{0.818mm}
\pgfsetroundcap
\pgfsetroundjoin
\pgfsetmiterlimit{7.0}
\definecolor{eps2pgf_color}{gray}{0}\pgfsetstrokecolor{eps2pgf_color}\pgfsetfillcolor{eps2pgf_color}
\pgfpathmoveto{\pgfqpoint{0.117cm}{1.421cm}}
\pgfpathlineto{\pgfqpoint{0.682cm}{0.671cm}}
\pgfpathlineto{\pgfqpoint{1.246cm}{1.421cm}}
\pgfusepath{stroke}
\end{pgfscope}
\definecolor{eps2pgf_color}{gray}{0}\pgfsetstrokecolor{eps2pgf_color}\pgfsetfillcolor{eps2pgf_color}
\pgfpathmoveto{\pgfqpoint{0.273cm}{1.395cm}}
\pgfpathcurveto{\pgfqpoint{0.273cm}{1.432cm}}{\pgfqpoint{0.259cm}{1.467cm}}{\pgfqpoint{0.233cm}{1.492cm}}
\pgfpathcurveto{\pgfqpoint{0.207cm}{1.518cm}}{\pgfqpoint{0.173cm}{1.532cm}}{\pgfqpoint{0.137cm}{1.532cm}}
\pgfpathcurveto{\pgfqpoint{0.1cm}{1.532cm}}{\pgfqpoint{0.066cm}{1.518cm}}{\pgfqpoint{0.04cm}{1.492cm}}
\pgfpathcurveto{\pgfqpoint{0.014cm}{1.467cm}}{\pgfqpoint{0cm}{1.432cm}}{\pgfqpoint{0cm}{1.395cm}}
\pgfpathcurveto{\pgfqpoint{0cm}{1.359cm}}{\pgfqpoint{0.014cm}{1.324cm}}{\pgfqpoint{0.04cm}{1.299cm}}
\pgfpathcurveto{\pgfqpoint{0.066cm}{1.273cm}}{\pgfqpoint{0.1cm}{1.258cm}}{\pgfqpoint{0.137cm}{1.258cm}}
\pgfpathcurveto{\pgfqpoint{0.173cm}{1.258cm}}{\pgfqpoint{0.207cm}{1.273cm}}{\pgfqpoint{0.233cm}{1.299cm}}
\pgfpathcurveto{\pgfqpoint{0.259cm}{1.324cm}}{\pgfqpoint{0.273cm}{1.359cm}}{\pgfqpoint{0.273cm}{1.395cm}}
\pgfusepath{fill}
\begin{pgfscope}
\pgfsetdash{}{0cm}
\pgfsetlinewidth{0.818mm}
\pgfsetmiterlimit{7.0}
\pgfpathmoveto{\pgfqpoint{0.682cm}{0.671cm}}
\pgfpathlineto{\pgfqpoint{0.679cm}{1.418cm}}
\pgfusepath{stroke}
\end{pgfscope}
\pgfpathmoveto{\pgfqpoint{0.815cm}{1.399cm}}
\pgfpathcurveto{\pgfqpoint{0.815cm}{1.435cm}}{\pgfqpoint{0.801cm}{1.47cm}}{\pgfqpoint{0.775cm}{1.496cm}}
\pgfpathcurveto{\pgfqpoint{0.75cm}{1.521cm}}{\pgfqpoint{0.715cm}{1.536cm}}{\pgfqpoint{0.679cm}{1.536cm}}
\pgfpathcurveto{\pgfqpoint{0.643cm}{1.536cm}}{\pgfqpoint{0.608cm}{1.521cm}}{\pgfqpoint{0.582cm}{1.496cm}}
\pgfpathcurveto{\pgfqpoint{0.557cm}{1.47cm}}{\pgfqpoint{0.542cm}{1.435cm}}{\pgfqpoint{0.542cm}{1.399cm}}
\pgfpathcurveto{\pgfqpoint{0.542cm}{1.363cm}}{\pgfqpoint{0.557cm}{1.328cm}}{\pgfqpoint{0.582cm}{1.302cm}}
\pgfpathcurveto{\pgfqpoint{0.608cm}{1.276cm}}{\pgfqpoint{0.643cm}{1.262cm}}{\pgfqpoint{0.679cm}{1.262cm}}
\pgfpathcurveto{\pgfqpoint{0.715cm}{1.262cm}}{\pgfqpoint{0.75cm}{1.276cm}}{\pgfqpoint{0.775cm}{1.302cm}}
\pgfpathcurveto{\pgfqpoint{0.801cm}{1.328cm}}{\pgfqpoint{0.815cm}{1.363cm}}{\pgfqpoint{0.815cm}{1.399cm}}
\pgfusepath{fill}
\pgfpathmoveto{\pgfqpoint{1.345cm}{1.371cm}}
\pgfpathcurveto{\pgfqpoint{1.345cm}{1.408cm}}{\pgfqpoint{1.331cm}{1.442cm}}{\pgfqpoint{1.305cm}{1.468cm}}
\pgfpathcurveto{\pgfqpoint{1.28cm}{1.494cm}}{\pgfqpoint{1.245cm}{1.508cm}}{\pgfqpoint{1.209cm}{1.508cm}}
\pgfpathcurveto{\pgfqpoint{1.172cm}{1.508cm}}{\pgfqpoint{1.138cm}{1.494cm}}{\pgfqpoint{1.112cm}{1.468cm}}
\pgfpathcurveto{\pgfqpoint{1.087cm}{1.442cm}}{\pgfqpoint{1.072cm}{1.408cm}}{\pgfqpoint{1.072cm}{1.371cm}}
\pgfpathcurveto{\pgfqpoint{1.072cm}{1.335cm}}{\pgfqpoint{1.087cm}{1.3cm}}{\pgfqpoint{1.112cm}{1.274cm}}
\pgfpathcurveto{\pgfqpoint{1.138cm}{1.249cm}}{\pgfqpoint{1.172cm}{1.234cm}}{\pgfqpoint{1.209cm}{1.234cm}}
\pgfpathcurveto{\pgfqpoint{1.245cm}{1.234cm}}{\pgfqpoint{1.28cm}{1.249cm}}{\pgfqpoint{1.305cm}{1.274cm}}
\pgfpathcurveto{\pgfqpoint{1.331cm}{1.3cm}}{\pgfqpoint{1.345cm}{1.335cm}}{\pgfqpoint{1.345cm}{1.371cm}}
\pgfusepath{fill}
\begin{pgfscope}
\pgfsetdash{}{0cm}
\pgfsetlinewidth{0.818mm}
\pgfsetroundcap
\pgfsetmiterlimit{4.0}
\pgfpathmoveto{\pgfqpoint{0.682cm}{0.671cm}}
\pgfpathlineto{\pgfqpoint{0.682cm}{0.042cm}}
\pgfusepath{stroke}
\end{pgfscope}
\end{pgfscope}
\end{pgfscope}
\end{pgfscope}
\end{tikzpicture}}} \|_{C_T^{\beta / 2}
     L^{\infty, \varepsilon} (\rho^{\sigma})}], \]
  where the proportionality constant is independent of $M,\varepsilon$.
\end{lemma}

\begin{proof}
  We define a fixed point map 
  \[
  \mathcal{K} : \tilde{Y} \mapsto Y
  \assign -\lambda X_{M, \varepsilon}^{\!\resizebox{0.6em}{!}{
\begin{tikzpicture}
\pgfpathmoveto{\pgfqpoint{0cm}{-0.035cm}}
\pgfpathlineto{\pgfqpoint{1.376cm}{-0.035cm}}
\pgfpathlineto{\pgfqpoint{1.376cm}{1.552cm}}
\pgfpathlineto{\pgfqpoint{0cm}{1.552cm}}
\pgfpathclose
\pgfusepath{clip}
\begin{pgfscope}
\begin{pgfscope}
\pgfpathmoveto{\pgfqpoint{0cm}{-0.035cm}}
\pgfpathlineto{\pgfqpoint{1.376cm}{-0.035cm}}
\pgfpathlineto{\pgfqpoint{1.376cm}{1.552cm}}
\pgfpathlineto{\pgfqpoint{0cm}{1.552cm}}
\pgfpathclose
\pgfusepath{clip}
\begin{pgfscope}
\begin{pgfscope}
\pgfsetdash{}{0cm}
\pgfsetlinewidth{0.818mm}
\pgfsetroundcap
\pgfsetroundjoin
\pgfsetmiterlimit{7.0}
\definecolor{eps2pgf_color}{gray}{0}\pgfsetstrokecolor{eps2pgf_color}\pgfsetfillcolor{eps2pgf_color}
\pgfpathmoveto{\pgfqpoint{0.117cm}{1.421cm}}
\pgfpathlineto{\pgfqpoint{0.682cm}{0.671cm}}
\pgfpathlineto{\pgfqpoint{1.246cm}{1.421cm}}
\pgfusepath{stroke}
\end{pgfscope}
\definecolor{eps2pgf_color}{gray}{0}\pgfsetstrokecolor{eps2pgf_color}\pgfsetfillcolor{eps2pgf_color}
\pgfpathmoveto{\pgfqpoint{0.273cm}{1.395cm}}
\pgfpathcurveto{\pgfqpoint{0.273cm}{1.432cm}}{\pgfqpoint{0.259cm}{1.467cm}}{\pgfqpoint{0.233cm}{1.492cm}}
\pgfpathcurveto{\pgfqpoint{0.207cm}{1.518cm}}{\pgfqpoint{0.173cm}{1.532cm}}{\pgfqpoint{0.137cm}{1.532cm}}
\pgfpathcurveto{\pgfqpoint{0.1cm}{1.532cm}}{\pgfqpoint{0.066cm}{1.518cm}}{\pgfqpoint{0.04cm}{1.492cm}}
\pgfpathcurveto{\pgfqpoint{0.014cm}{1.467cm}}{\pgfqpoint{0cm}{1.432cm}}{\pgfqpoint{0cm}{1.395cm}}
\pgfpathcurveto{\pgfqpoint{0cm}{1.359cm}}{\pgfqpoint{0.014cm}{1.324cm}}{\pgfqpoint{0.04cm}{1.299cm}}
\pgfpathcurveto{\pgfqpoint{0.066cm}{1.273cm}}{\pgfqpoint{0.1cm}{1.258cm}}{\pgfqpoint{0.137cm}{1.258cm}}
\pgfpathcurveto{\pgfqpoint{0.173cm}{1.258cm}}{\pgfqpoint{0.207cm}{1.273cm}}{\pgfqpoint{0.233cm}{1.299cm}}
\pgfpathcurveto{\pgfqpoint{0.259cm}{1.324cm}}{\pgfqpoint{0.273cm}{1.359cm}}{\pgfqpoint{0.273cm}{1.395cm}}
\pgfusepath{fill}
\begin{pgfscope}
\pgfsetdash{}{0cm}
\pgfsetlinewidth{0.818mm}
\pgfsetmiterlimit{7.0}
\pgfpathmoveto{\pgfqpoint{0.682cm}{0.671cm}}
\pgfpathlineto{\pgfqpoint{0.679cm}{1.418cm}}
\pgfusepath{stroke}
\end{pgfscope}
\pgfpathmoveto{\pgfqpoint{0.815cm}{1.399cm}}
\pgfpathcurveto{\pgfqpoint{0.815cm}{1.435cm}}{\pgfqpoint{0.801cm}{1.47cm}}{\pgfqpoint{0.775cm}{1.496cm}}
\pgfpathcurveto{\pgfqpoint{0.75cm}{1.521cm}}{\pgfqpoint{0.715cm}{1.536cm}}{\pgfqpoint{0.679cm}{1.536cm}}
\pgfpathcurveto{\pgfqpoint{0.643cm}{1.536cm}}{\pgfqpoint{0.608cm}{1.521cm}}{\pgfqpoint{0.582cm}{1.496cm}}
\pgfpathcurveto{\pgfqpoint{0.557cm}{1.47cm}}{\pgfqpoint{0.542cm}{1.435cm}}{\pgfqpoint{0.542cm}{1.399cm}}
\pgfpathcurveto{\pgfqpoint{0.542cm}{1.363cm}}{\pgfqpoint{0.557cm}{1.328cm}}{\pgfqpoint{0.582cm}{1.302cm}}
\pgfpathcurveto{\pgfqpoint{0.608cm}{1.276cm}}{\pgfqpoint{0.643cm}{1.262cm}}{\pgfqpoint{0.679cm}{1.262cm}}
\pgfpathcurveto{\pgfqpoint{0.715cm}{1.262cm}}{\pgfqpoint{0.75cm}{1.276cm}}{\pgfqpoint{0.775cm}{1.302cm}}
\pgfpathcurveto{\pgfqpoint{0.801cm}{1.328cm}}{\pgfqpoint{0.815cm}{1.363cm}}{\pgfqpoint{0.815cm}{1.399cm}}
\pgfusepath{fill}
\pgfpathmoveto{\pgfqpoint{1.345cm}{1.371cm}}
\pgfpathcurveto{\pgfqpoint{1.345cm}{1.408cm}}{\pgfqpoint{1.331cm}{1.442cm}}{\pgfqpoint{1.305cm}{1.468cm}}
\pgfpathcurveto{\pgfqpoint{1.28cm}{1.494cm}}{\pgfqpoint{1.245cm}{1.508cm}}{\pgfqpoint{1.209cm}{1.508cm}}
\pgfpathcurveto{\pgfqpoint{1.172cm}{1.508cm}}{\pgfqpoint{1.138cm}{1.494cm}}{\pgfqpoint{1.112cm}{1.468cm}}
\pgfpathcurveto{\pgfqpoint{1.087cm}{1.442cm}}{\pgfqpoint{1.072cm}{1.408cm}}{\pgfqpoint{1.072cm}{1.371cm}}
\pgfpathcurveto{\pgfqpoint{1.072cm}{1.335cm}}{\pgfqpoint{1.087cm}{1.3cm}}{\pgfqpoint{1.112cm}{1.274cm}}
\pgfpathcurveto{\pgfqpoint{1.138cm}{1.249cm}}{\pgfqpoint{1.172cm}{1.234cm}}{\pgfqpoint{1.209cm}{1.234cm}}
\pgfpathcurveto{\pgfqpoint{1.245cm}{1.234cm}}{\pgfqpoint{1.28cm}{1.249cm}}{\pgfqpoint{1.305cm}{1.274cm}}
\pgfpathcurveto{\pgfqpoint{1.331cm}{1.3cm}}{\pgfqpoint{1.345cm}{1.335cm}}{\pgfqpoint{1.345cm}{1.371cm}}
\pgfusepath{fill}
\begin{pgfscope}
\pgfsetdash{}{0cm}
\pgfsetlinewidth{0.818mm}
\pgfsetroundcap
\pgfsetmiterlimit{4.0}
\pgfpathmoveto{\pgfqpoint{0.682cm}{0.671cm}}
\pgfpathlineto{\pgfqpoint{0.682cm}{0.042cm}}
\pgfusepath{stroke}
\end{pgfscope}
\end{pgfscope}
\end{pgfscope}
\end{pgfscope}
\end{tikzpicture}}} - \LL_{\varepsilon}^{- 1} [ 3\lambda
  ( \UU^{\varepsilon}_{>} \llbracket X_{M, \varepsilon}^2 \rrbracket ) \succ
  \tilde{Y} ]
  \] for some $L > 0$ to be chosen below. Then in
  view of the Schauder estimates from Lemma~3.4 in {\cite{MP17}}, the
  paraproduct estimates as well as Lemma~\ref{lem:loc}, we have
  \[ \| \mathcal{K} \tilde{Y}_1 - \mathcal{K} \tilde{Y}_2 \|_{C_T \CC^{1 / 2 -
     \kappa, \varepsilon} (\rho^{\sigma})} \lesssim \lambda \| (
     \UU^{\varepsilon}_{>} \llbracket X_{M, \varepsilon}^2 \rrbracket )
     \succ (\tilde{Y_1} - \tilde{Y_2}) \|_{C_T \CC^{- 3 / 2 -
     \kappa, \varepsilon} (\rho^{\sigma})} \]
  \[ \leqslant C \lambda 2^{- L / 2} \| \llbracket X_{M, \varepsilon}^2 \rrbracket
     \|_{C_T \CC^{- 1 - \kappa, \varepsilon} (\rho^{\sigma})} \|
     \tilde{Y_1} - \tilde{Y_2} \|_{C_T L^{\infty, \varepsilon}
     (\rho^{\sigma})} \leqslant \delta \| \tilde{Y_1} - \tilde{Y_2}
     \|_{C_T \CC^{1 / 2 - \kappa, \varepsilon} (\rho^{\sigma})} \]
  for some $\delta \in (0, 1)$ independent of $\lambda, M,\varepsilon$ provided $L=L(\lambda, M,\varepsilon)$ in the definition of the localizer
  $\UU^{\varepsilon}_{>}$ is chosen to be the smallest $L\geqslant 0$ such that
  \[ \lambda \left\| \UU^{\varepsilon}_{>} \llbracket X_{M, \varepsilon}^2 \rrbracket
     \right\|_{C_T \CC^{- 3 / 2 - \kappa, \varepsilon} (\rho^0)} \leqslant C \lambda
     2^{- L / 2} \| \llbracket X_{M, \varepsilon}^2 \rrbracket \|_{C_T \CC^{-
     1 - \kappa, \varepsilon} (\rho^{\sigma})} \leqslant \delta . \]
In particular, we have that
 \begin{equation}
 2^{L/2}= C_{\delta}( 1+\lambda \| \llbracket X_{M, \varepsilon}^2 \rrbracket \|_{C_T \CC^{-
     1 - \kappa, \varepsilon} (\rho^{\sigma})}),
    \label{eq:U11}
  \end{equation}
  which will be used later in order to estimate the complementary operator $\UU^{\varepsilon}_{\leqslant}$ by Lemma~\ref{lem:loc}.
Note that   $L(\lambda,{M,\varepsilon})$ a~priori depends on $M,\varepsilon$. However, due to the uniform bound on
$$\|\llbracket X^{2}_{M,\varepsilon}\rrbracket\|_{C_{T}\CC^{-1-\kappa/2,\varepsilon}(\rho^{\sigma})}+\|\llbracket X^{2}_{M,\varepsilon}\rrbracket\|_{C^{\gamma/2}_{T}L^{\infty,\varepsilon}(\rho^{\sigma})}$$ valid for some $\gamma\in (0,1)$, we may use compactness to deduce that for every fixed $\lambda>0$ there exists a subsequence (not relabeled) such that $L(\lambda,M,\varepsilon)\to L_{0}(\lambda)$. This will also allow to identify the limit of  the localized term below  in Section \ref{s:sd}.

Next, we estimate
  \[ \| \mathcal{K} \tilde{Y} \|_{C_T \CC^{1 / 2 - \kappa, \varepsilon}
     (\rho^{\sigma})} \leqslant  \lambda \| X_{M, \varepsilon}^{\!\resizebox{0.6em}{!}{
\begin{tikzpicture}
\pgfpathmoveto{\pgfqpoint{0cm}{-0.035cm}}
\pgfpathlineto{\pgfqpoint{1.376cm}{-0.035cm}}
\pgfpathlineto{\pgfqpoint{1.376cm}{1.552cm}}
\pgfpathlineto{\pgfqpoint{0cm}{1.552cm}}
\pgfpathclose
\pgfusepath{clip}
\begin{pgfscope}
\begin{pgfscope}
\pgfpathmoveto{\pgfqpoint{0cm}{-0.035cm}}
\pgfpathlineto{\pgfqpoint{1.376cm}{-0.035cm}}
\pgfpathlineto{\pgfqpoint{1.376cm}{1.552cm}}
\pgfpathlineto{\pgfqpoint{0cm}{1.552cm}}
\pgfpathclose
\pgfusepath{clip}
\begin{pgfscope}
\begin{pgfscope}
\pgfsetdash{}{0cm}
\pgfsetlinewidth{0.818mm}
\pgfsetroundcap
\pgfsetroundjoin
\pgfsetmiterlimit{7.0}
\definecolor{eps2pgf_color}{gray}{0}\pgfsetstrokecolor{eps2pgf_color}\pgfsetfillcolor{eps2pgf_color}
\pgfpathmoveto{\pgfqpoint{0.117cm}{1.421cm}}
\pgfpathlineto{\pgfqpoint{0.682cm}{0.671cm}}
\pgfpathlineto{\pgfqpoint{1.246cm}{1.421cm}}
\pgfusepath{stroke}
\end{pgfscope}
\definecolor{eps2pgf_color}{gray}{0}\pgfsetstrokecolor{eps2pgf_color}\pgfsetfillcolor{eps2pgf_color}
\pgfpathmoveto{\pgfqpoint{0.273cm}{1.395cm}}
\pgfpathcurveto{\pgfqpoint{0.273cm}{1.432cm}}{\pgfqpoint{0.259cm}{1.467cm}}{\pgfqpoint{0.233cm}{1.492cm}}
\pgfpathcurveto{\pgfqpoint{0.207cm}{1.518cm}}{\pgfqpoint{0.173cm}{1.532cm}}{\pgfqpoint{0.137cm}{1.532cm}}
\pgfpathcurveto{\pgfqpoint{0.1cm}{1.532cm}}{\pgfqpoint{0.066cm}{1.518cm}}{\pgfqpoint{0.04cm}{1.492cm}}
\pgfpathcurveto{\pgfqpoint{0.014cm}{1.467cm}}{\pgfqpoint{0cm}{1.432cm}}{\pgfqpoint{0cm}{1.395cm}}
\pgfpathcurveto{\pgfqpoint{0cm}{1.359cm}}{\pgfqpoint{0.014cm}{1.324cm}}{\pgfqpoint{0.04cm}{1.299cm}}
\pgfpathcurveto{\pgfqpoint{0.066cm}{1.273cm}}{\pgfqpoint{0.1cm}{1.258cm}}{\pgfqpoint{0.137cm}{1.258cm}}
\pgfpathcurveto{\pgfqpoint{0.173cm}{1.258cm}}{\pgfqpoint{0.207cm}{1.273cm}}{\pgfqpoint{0.233cm}{1.299cm}}
\pgfpathcurveto{\pgfqpoint{0.259cm}{1.324cm}}{\pgfqpoint{0.273cm}{1.359cm}}{\pgfqpoint{0.273cm}{1.395cm}}
\pgfusepath{fill}
\begin{pgfscope}
\pgfsetdash{}{0cm}
\pgfsetlinewidth{0.818mm}
\pgfsetmiterlimit{7.0}
\pgfpathmoveto{\pgfqpoint{0.682cm}{0.671cm}}
\pgfpathlineto{\pgfqpoint{0.679cm}{1.418cm}}
\pgfusepath{stroke}
\end{pgfscope}
\pgfpathmoveto{\pgfqpoint{0.815cm}{1.399cm}}
\pgfpathcurveto{\pgfqpoint{0.815cm}{1.435cm}}{\pgfqpoint{0.801cm}{1.47cm}}{\pgfqpoint{0.775cm}{1.496cm}}
\pgfpathcurveto{\pgfqpoint{0.75cm}{1.521cm}}{\pgfqpoint{0.715cm}{1.536cm}}{\pgfqpoint{0.679cm}{1.536cm}}
\pgfpathcurveto{\pgfqpoint{0.643cm}{1.536cm}}{\pgfqpoint{0.608cm}{1.521cm}}{\pgfqpoint{0.582cm}{1.496cm}}
\pgfpathcurveto{\pgfqpoint{0.557cm}{1.47cm}}{\pgfqpoint{0.542cm}{1.435cm}}{\pgfqpoint{0.542cm}{1.399cm}}
\pgfpathcurveto{\pgfqpoint{0.542cm}{1.363cm}}{\pgfqpoint{0.557cm}{1.328cm}}{\pgfqpoint{0.582cm}{1.302cm}}
\pgfpathcurveto{\pgfqpoint{0.608cm}{1.276cm}}{\pgfqpoint{0.643cm}{1.262cm}}{\pgfqpoint{0.679cm}{1.262cm}}
\pgfpathcurveto{\pgfqpoint{0.715cm}{1.262cm}}{\pgfqpoint{0.75cm}{1.276cm}}{\pgfqpoint{0.775cm}{1.302cm}}
\pgfpathcurveto{\pgfqpoint{0.801cm}{1.328cm}}{\pgfqpoint{0.815cm}{1.363cm}}{\pgfqpoint{0.815cm}{1.399cm}}
\pgfusepath{fill}
\pgfpathmoveto{\pgfqpoint{1.345cm}{1.371cm}}
\pgfpathcurveto{\pgfqpoint{1.345cm}{1.408cm}}{\pgfqpoint{1.331cm}{1.442cm}}{\pgfqpoint{1.305cm}{1.468cm}}
\pgfpathcurveto{\pgfqpoint{1.28cm}{1.494cm}}{\pgfqpoint{1.245cm}{1.508cm}}{\pgfqpoint{1.209cm}{1.508cm}}
\pgfpathcurveto{\pgfqpoint{1.172cm}{1.508cm}}{\pgfqpoint{1.138cm}{1.494cm}}{\pgfqpoint{1.112cm}{1.468cm}}
\pgfpathcurveto{\pgfqpoint{1.087cm}{1.442cm}}{\pgfqpoint{1.072cm}{1.408cm}}{\pgfqpoint{1.072cm}{1.371cm}}
\pgfpathcurveto{\pgfqpoint{1.072cm}{1.335cm}}{\pgfqpoint{1.087cm}{1.3cm}}{\pgfqpoint{1.112cm}{1.274cm}}
\pgfpathcurveto{\pgfqpoint{1.138cm}{1.249cm}}{\pgfqpoint{1.172cm}{1.234cm}}{\pgfqpoint{1.209cm}{1.234cm}}
\pgfpathcurveto{\pgfqpoint{1.245cm}{1.234cm}}{\pgfqpoint{1.28cm}{1.249cm}}{\pgfqpoint{1.305cm}{1.274cm}}
\pgfpathcurveto{\pgfqpoint{1.331cm}{1.3cm}}{\pgfqpoint{1.345cm}{1.335cm}}{\pgfqpoint{1.345cm}{1.371cm}}
\pgfusepath{fill}
\begin{pgfscope}
\pgfsetdash{}{0cm}
\pgfsetlinewidth{0.818mm}
\pgfsetroundcap
\pgfsetmiterlimit{4.0}
\pgfpathmoveto{\pgfqpoint{0.682cm}{0.671cm}}
\pgfpathlineto{\pgfqpoint{0.682cm}{0.042cm}}
\pgfusepath{stroke}
\end{pgfscope}
\end{pgfscope}
\end{pgfscope}
\end{pgfscope}
\end{tikzpicture}}}
     \|_{C_T \CC^{1 / 2 - \kappa, \varepsilon} (\rho^{\sigma})} + C \lambda
     \| ( \UU^{\varepsilon}_{>} \llbracket X_{M, \varepsilon}^2
     \rrbracket ) \succ \tilde{Y} \|_{C_T \CC^{- 3 / 2 - \kappa,
     \varepsilon} (\rho^{\sigma})} \]
  \[ \leqslant \lambda \| X_{M, \varepsilon}^{\!\resizebox{0.6em}{!}{
\begin{tikzpicture}
\pgfpathmoveto{\pgfqpoint{0cm}{-0.035cm}}
\pgfpathlineto{\pgfqpoint{1.376cm}{-0.035cm}}
\pgfpathlineto{\pgfqpoint{1.376cm}{1.552cm}}
\pgfpathlineto{\pgfqpoint{0cm}{1.552cm}}
\pgfpathclose
\pgfusepath{clip}
\begin{pgfscope}
\begin{pgfscope}
\pgfpathmoveto{\pgfqpoint{0cm}{-0.035cm}}
\pgfpathlineto{\pgfqpoint{1.376cm}{-0.035cm}}
\pgfpathlineto{\pgfqpoint{1.376cm}{1.552cm}}
\pgfpathlineto{\pgfqpoint{0cm}{1.552cm}}
\pgfpathclose
\pgfusepath{clip}
\begin{pgfscope}
\begin{pgfscope}
\pgfsetdash{}{0cm}
\pgfsetlinewidth{0.818mm}
\pgfsetroundcap
\pgfsetroundjoin
\pgfsetmiterlimit{7.0}
\definecolor{eps2pgf_color}{gray}{0}\pgfsetstrokecolor{eps2pgf_color}\pgfsetfillcolor{eps2pgf_color}
\pgfpathmoveto{\pgfqpoint{0.117cm}{1.421cm}}
\pgfpathlineto{\pgfqpoint{0.682cm}{0.671cm}}
\pgfpathlineto{\pgfqpoint{1.246cm}{1.421cm}}
\pgfusepath{stroke}
\end{pgfscope}
\definecolor{eps2pgf_color}{gray}{0}\pgfsetstrokecolor{eps2pgf_color}\pgfsetfillcolor{eps2pgf_color}
\pgfpathmoveto{\pgfqpoint{0.273cm}{1.395cm}}
\pgfpathcurveto{\pgfqpoint{0.273cm}{1.432cm}}{\pgfqpoint{0.259cm}{1.467cm}}{\pgfqpoint{0.233cm}{1.492cm}}
\pgfpathcurveto{\pgfqpoint{0.207cm}{1.518cm}}{\pgfqpoint{0.173cm}{1.532cm}}{\pgfqpoint{0.137cm}{1.532cm}}
\pgfpathcurveto{\pgfqpoint{0.1cm}{1.532cm}}{\pgfqpoint{0.066cm}{1.518cm}}{\pgfqpoint{0.04cm}{1.492cm}}
\pgfpathcurveto{\pgfqpoint{0.014cm}{1.467cm}}{\pgfqpoint{0cm}{1.432cm}}{\pgfqpoint{0cm}{1.395cm}}
\pgfpathcurveto{\pgfqpoint{0cm}{1.359cm}}{\pgfqpoint{0.014cm}{1.324cm}}{\pgfqpoint{0.04cm}{1.299cm}}
\pgfpathcurveto{\pgfqpoint{0.066cm}{1.273cm}}{\pgfqpoint{0.1cm}{1.258cm}}{\pgfqpoint{0.137cm}{1.258cm}}
\pgfpathcurveto{\pgfqpoint{0.173cm}{1.258cm}}{\pgfqpoint{0.207cm}{1.273cm}}{\pgfqpoint{0.233cm}{1.299cm}}
\pgfpathcurveto{\pgfqpoint{0.259cm}{1.324cm}}{\pgfqpoint{0.273cm}{1.359cm}}{\pgfqpoint{0.273cm}{1.395cm}}
\pgfusepath{fill}
\begin{pgfscope}
\pgfsetdash{}{0cm}
\pgfsetlinewidth{0.818mm}
\pgfsetmiterlimit{7.0}
\pgfpathmoveto{\pgfqpoint{0.682cm}{0.671cm}}
\pgfpathlineto{\pgfqpoint{0.679cm}{1.418cm}}
\pgfusepath{stroke}
\end{pgfscope}
\pgfpathmoveto{\pgfqpoint{0.815cm}{1.399cm}}
\pgfpathcurveto{\pgfqpoint{0.815cm}{1.435cm}}{\pgfqpoint{0.801cm}{1.47cm}}{\pgfqpoint{0.775cm}{1.496cm}}
\pgfpathcurveto{\pgfqpoint{0.75cm}{1.521cm}}{\pgfqpoint{0.715cm}{1.536cm}}{\pgfqpoint{0.679cm}{1.536cm}}
\pgfpathcurveto{\pgfqpoint{0.643cm}{1.536cm}}{\pgfqpoint{0.608cm}{1.521cm}}{\pgfqpoint{0.582cm}{1.496cm}}
\pgfpathcurveto{\pgfqpoint{0.557cm}{1.47cm}}{\pgfqpoint{0.542cm}{1.435cm}}{\pgfqpoint{0.542cm}{1.399cm}}
\pgfpathcurveto{\pgfqpoint{0.542cm}{1.363cm}}{\pgfqpoint{0.557cm}{1.328cm}}{\pgfqpoint{0.582cm}{1.302cm}}
\pgfpathcurveto{\pgfqpoint{0.608cm}{1.276cm}}{\pgfqpoint{0.643cm}{1.262cm}}{\pgfqpoint{0.679cm}{1.262cm}}
\pgfpathcurveto{\pgfqpoint{0.715cm}{1.262cm}}{\pgfqpoint{0.75cm}{1.276cm}}{\pgfqpoint{0.775cm}{1.302cm}}
\pgfpathcurveto{\pgfqpoint{0.801cm}{1.328cm}}{\pgfqpoint{0.815cm}{1.363cm}}{\pgfqpoint{0.815cm}{1.399cm}}
\pgfusepath{fill}
\pgfpathmoveto{\pgfqpoint{1.345cm}{1.371cm}}
\pgfpathcurveto{\pgfqpoint{1.345cm}{1.408cm}}{\pgfqpoint{1.331cm}{1.442cm}}{\pgfqpoint{1.305cm}{1.468cm}}
\pgfpathcurveto{\pgfqpoint{1.28cm}{1.494cm}}{\pgfqpoint{1.245cm}{1.508cm}}{\pgfqpoint{1.209cm}{1.508cm}}
\pgfpathcurveto{\pgfqpoint{1.172cm}{1.508cm}}{\pgfqpoint{1.138cm}{1.494cm}}{\pgfqpoint{1.112cm}{1.468cm}}
\pgfpathcurveto{\pgfqpoint{1.087cm}{1.442cm}}{\pgfqpoint{1.072cm}{1.408cm}}{\pgfqpoint{1.072cm}{1.371cm}}
\pgfpathcurveto{\pgfqpoint{1.072cm}{1.335cm}}{\pgfqpoint{1.087cm}{1.3cm}}{\pgfqpoint{1.112cm}{1.274cm}}
\pgfpathcurveto{\pgfqpoint{1.138cm}{1.249cm}}{\pgfqpoint{1.172cm}{1.234cm}}{\pgfqpoint{1.209cm}{1.234cm}}
\pgfpathcurveto{\pgfqpoint{1.245cm}{1.234cm}}{\pgfqpoint{1.28cm}{1.249cm}}{\pgfqpoint{1.305cm}{1.274cm}}
\pgfpathcurveto{\pgfqpoint{1.331cm}{1.3cm}}{\pgfqpoint{1.345cm}{1.335cm}}{\pgfqpoint{1.345cm}{1.371cm}}
\pgfusepath{fill}
\begin{pgfscope}
\pgfsetdash{}{0cm}
\pgfsetlinewidth{0.818mm}
\pgfsetroundcap
\pgfsetmiterlimit{4.0}
\pgfpathmoveto{\pgfqpoint{0.682cm}{0.671cm}}
\pgfpathlineto{\pgfqpoint{0.682cm}{0.042cm}}
\pgfusepath{stroke}
\end{pgfscope}
\end{pgfscope}
\end{pgfscope}
\end{pgfscope}
\end{tikzpicture}}} \|_{C_T \CC^{1 / 2
     - \kappa, \varepsilon} (\rho^{\sigma})} +  \delta \| \tilde{Y} \|_{C_T
     \CC^{1 / 2 - \kappa, \varepsilon} (\rho^{\sigma})} . \]
  Therefore we deduce that $\mathcal{K}$ leaves balls in $C_T \CC^{1 / 2 -
  \kappa, \varepsilon} (\rho^{\sigma})$ invariant and is a contraction on $C_T
  \CC^{1 / 2 - \kappa, \varepsilon} (\rho^{\sigma})$. Hence there exists a
  unique fixed point $Y_{M, \varepsilon}$ and the first bound follows. Next,
  we use the Schauder estimates (see Lemma 3.10 in {\cite{MP17}}) to bound the
  time regularity as follows
  \[ \| Y_{M, \varepsilon} \|_{C_T^{\beta / 2} L^{\infty, \varepsilon}
     (\rho^{\sigma})} \leqslant \lambda \| X_{M, \varepsilon}^{\!\resizebox{0.6em}{!}{
\begin{tikzpicture}
\pgfpathmoveto{\pgfqpoint{0cm}{-0.035cm}}
\pgfpathlineto{\pgfqpoint{1.376cm}{-0.035cm}}
\pgfpathlineto{\pgfqpoint{1.376cm}{1.552cm}}
\pgfpathlineto{\pgfqpoint{0cm}{1.552cm}}
\pgfpathclose
\pgfusepath{clip}
\begin{pgfscope}
\begin{pgfscope}
\pgfpathmoveto{\pgfqpoint{0cm}{-0.035cm}}
\pgfpathlineto{\pgfqpoint{1.376cm}{-0.035cm}}
\pgfpathlineto{\pgfqpoint{1.376cm}{1.552cm}}
\pgfpathlineto{\pgfqpoint{0cm}{1.552cm}}
\pgfpathclose
\pgfusepath{clip}
\begin{pgfscope}
\begin{pgfscope}
\pgfsetdash{}{0cm}
\pgfsetlinewidth{0.818mm}
\pgfsetroundcap
\pgfsetroundjoin
\pgfsetmiterlimit{7.0}
\definecolor{eps2pgf_color}{gray}{0}\pgfsetstrokecolor{eps2pgf_color}\pgfsetfillcolor{eps2pgf_color}
\pgfpathmoveto{\pgfqpoint{0.117cm}{1.421cm}}
\pgfpathlineto{\pgfqpoint{0.682cm}{0.671cm}}
\pgfpathlineto{\pgfqpoint{1.246cm}{1.421cm}}
\pgfusepath{stroke}
\end{pgfscope}
\definecolor{eps2pgf_color}{gray}{0}\pgfsetstrokecolor{eps2pgf_color}\pgfsetfillcolor{eps2pgf_color}
\pgfpathmoveto{\pgfqpoint{0.273cm}{1.395cm}}
\pgfpathcurveto{\pgfqpoint{0.273cm}{1.432cm}}{\pgfqpoint{0.259cm}{1.467cm}}{\pgfqpoint{0.233cm}{1.492cm}}
\pgfpathcurveto{\pgfqpoint{0.207cm}{1.518cm}}{\pgfqpoint{0.173cm}{1.532cm}}{\pgfqpoint{0.137cm}{1.532cm}}
\pgfpathcurveto{\pgfqpoint{0.1cm}{1.532cm}}{\pgfqpoint{0.066cm}{1.518cm}}{\pgfqpoint{0.04cm}{1.492cm}}
\pgfpathcurveto{\pgfqpoint{0.014cm}{1.467cm}}{\pgfqpoint{0cm}{1.432cm}}{\pgfqpoint{0cm}{1.395cm}}
\pgfpathcurveto{\pgfqpoint{0cm}{1.359cm}}{\pgfqpoint{0.014cm}{1.324cm}}{\pgfqpoint{0.04cm}{1.299cm}}
\pgfpathcurveto{\pgfqpoint{0.066cm}{1.273cm}}{\pgfqpoint{0.1cm}{1.258cm}}{\pgfqpoint{0.137cm}{1.258cm}}
\pgfpathcurveto{\pgfqpoint{0.173cm}{1.258cm}}{\pgfqpoint{0.207cm}{1.273cm}}{\pgfqpoint{0.233cm}{1.299cm}}
\pgfpathcurveto{\pgfqpoint{0.259cm}{1.324cm}}{\pgfqpoint{0.273cm}{1.359cm}}{\pgfqpoint{0.273cm}{1.395cm}}
\pgfusepath{fill}
\begin{pgfscope}
\pgfsetdash{}{0cm}
\pgfsetlinewidth{0.818mm}
\pgfsetmiterlimit{7.0}
\pgfpathmoveto{\pgfqpoint{0.682cm}{0.671cm}}
\pgfpathlineto{\pgfqpoint{0.679cm}{1.418cm}}
\pgfusepath{stroke}
\end{pgfscope}
\pgfpathmoveto{\pgfqpoint{0.815cm}{1.399cm}}
\pgfpathcurveto{\pgfqpoint{0.815cm}{1.435cm}}{\pgfqpoint{0.801cm}{1.47cm}}{\pgfqpoint{0.775cm}{1.496cm}}
\pgfpathcurveto{\pgfqpoint{0.75cm}{1.521cm}}{\pgfqpoint{0.715cm}{1.536cm}}{\pgfqpoint{0.679cm}{1.536cm}}
\pgfpathcurveto{\pgfqpoint{0.643cm}{1.536cm}}{\pgfqpoint{0.608cm}{1.521cm}}{\pgfqpoint{0.582cm}{1.496cm}}
\pgfpathcurveto{\pgfqpoint{0.557cm}{1.47cm}}{\pgfqpoint{0.542cm}{1.435cm}}{\pgfqpoint{0.542cm}{1.399cm}}
\pgfpathcurveto{\pgfqpoint{0.542cm}{1.363cm}}{\pgfqpoint{0.557cm}{1.328cm}}{\pgfqpoint{0.582cm}{1.302cm}}
\pgfpathcurveto{\pgfqpoint{0.608cm}{1.276cm}}{\pgfqpoint{0.643cm}{1.262cm}}{\pgfqpoint{0.679cm}{1.262cm}}
\pgfpathcurveto{\pgfqpoint{0.715cm}{1.262cm}}{\pgfqpoint{0.75cm}{1.276cm}}{\pgfqpoint{0.775cm}{1.302cm}}
\pgfpathcurveto{\pgfqpoint{0.801cm}{1.328cm}}{\pgfqpoint{0.815cm}{1.363cm}}{\pgfqpoint{0.815cm}{1.399cm}}
\pgfusepath{fill}
\pgfpathmoveto{\pgfqpoint{1.345cm}{1.371cm}}
\pgfpathcurveto{\pgfqpoint{1.345cm}{1.408cm}}{\pgfqpoint{1.331cm}{1.442cm}}{\pgfqpoint{1.305cm}{1.468cm}}
\pgfpathcurveto{\pgfqpoint{1.28cm}{1.494cm}}{\pgfqpoint{1.245cm}{1.508cm}}{\pgfqpoint{1.209cm}{1.508cm}}
\pgfpathcurveto{\pgfqpoint{1.172cm}{1.508cm}}{\pgfqpoint{1.138cm}{1.494cm}}{\pgfqpoint{1.112cm}{1.468cm}}
\pgfpathcurveto{\pgfqpoint{1.087cm}{1.442cm}}{\pgfqpoint{1.072cm}{1.408cm}}{\pgfqpoint{1.072cm}{1.371cm}}
\pgfpathcurveto{\pgfqpoint{1.072cm}{1.335cm}}{\pgfqpoint{1.087cm}{1.3cm}}{\pgfqpoint{1.112cm}{1.274cm}}
\pgfpathcurveto{\pgfqpoint{1.138cm}{1.249cm}}{\pgfqpoint{1.172cm}{1.234cm}}{\pgfqpoint{1.209cm}{1.234cm}}
\pgfpathcurveto{\pgfqpoint{1.245cm}{1.234cm}}{\pgfqpoint{1.28cm}{1.249cm}}{\pgfqpoint{1.305cm}{1.274cm}}
\pgfpathcurveto{\pgfqpoint{1.331cm}{1.3cm}}{\pgfqpoint{1.345cm}{1.335cm}}{\pgfqpoint{1.345cm}{1.371cm}}
\pgfusepath{fill}
\begin{pgfscope}
\pgfsetdash{}{0cm}
\pgfsetlinewidth{0.818mm}
\pgfsetroundcap
\pgfsetmiterlimit{4.0}
\pgfpathmoveto{\pgfqpoint{0.682cm}{0.671cm}}
\pgfpathlineto{\pgfqpoint{0.682cm}{0.042cm}}
\pgfusepath{stroke}
\end{pgfscope}
\end{pgfscope}
\end{pgfscope}
\end{pgfscope}
\end{tikzpicture}}}
     \|_{C_T^{\beta / 2} L^{\infty, \varepsilon} (\rho^{\sigma})} + C \lambda
     \| ( \UU^{\varepsilon}_{>} \llbracket X_{M, \varepsilon}^2
     \rrbracket ) \succ Y_{M, \varepsilon} \|_{C_T \CC^{- 3 / 2 -
     \kappa, \varepsilon} (\rho^{\sigma})} \]
  \[ \leqslant \lambda \| X_{M, \varepsilon}^{\!\resizebox{0.6em}{!}{
\begin{tikzpicture}
\pgfpathmoveto{\pgfqpoint{0cm}{-0.035cm}}
\pgfpathlineto{\pgfqpoint{1.376cm}{-0.035cm}}
\pgfpathlineto{\pgfqpoint{1.376cm}{1.552cm}}
\pgfpathlineto{\pgfqpoint{0cm}{1.552cm}}
\pgfpathclose
\pgfusepath{clip}
\begin{pgfscope}
\begin{pgfscope}
\pgfpathmoveto{\pgfqpoint{0cm}{-0.035cm}}
\pgfpathlineto{\pgfqpoint{1.376cm}{-0.035cm}}
\pgfpathlineto{\pgfqpoint{1.376cm}{1.552cm}}
\pgfpathlineto{\pgfqpoint{0cm}{1.552cm}}
\pgfpathclose
\pgfusepath{clip}
\begin{pgfscope}
\begin{pgfscope}
\pgfsetdash{}{0cm}
\pgfsetlinewidth{0.818mm}
\pgfsetroundcap
\pgfsetroundjoin
\pgfsetmiterlimit{7.0}
\definecolor{eps2pgf_color}{gray}{0}\pgfsetstrokecolor{eps2pgf_color}\pgfsetfillcolor{eps2pgf_color}
\pgfpathmoveto{\pgfqpoint{0.117cm}{1.421cm}}
\pgfpathlineto{\pgfqpoint{0.682cm}{0.671cm}}
\pgfpathlineto{\pgfqpoint{1.246cm}{1.421cm}}
\pgfusepath{stroke}
\end{pgfscope}
\definecolor{eps2pgf_color}{gray}{0}\pgfsetstrokecolor{eps2pgf_color}\pgfsetfillcolor{eps2pgf_color}
\pgfpathmoveto{\pgfqpoint{0.273cm}{1.395cm}}
\pgfpathcurveto{\pgfqpoint{0.273cm}{1.432cm}}{\pgfqpoint{0.259cm}{1.467cm}}{\pgfqpoint{0.233cm}{1.492cm}}
\pgfpathcurveto{\pgfqpoint{0.207cm}{1.518cm}}{\pgfqpoint{0.173cm}{1.532cm}}{\pgfqpoint{0.137cm}{1.532cm}}
\pgfpathcurveto{\pgfqpoint{0.1cm}{1.532cm}}{\pgfqpoint{0.066cm}{1.518cm}}{\pgfqpoint{0.04cm}{1.492cm}}
\pgfpathcurveto{\pgfqpoint{0.014cm}{1.467cm}}{\pgfqpoint{0cm}{1.432cm}}{\pgfqpoint{0cm}{1.395cm}}
\pgfpathcurveto{\pgfqpoint{0cm}{1.359cm}}{\pgfqpoint{0.014cm}{1.324cm}}{\pgfqpoint{0.04cm}{1.299cm}}
\pgfpathcurveto{\pgfqpoint{0.066cm}{1.273cm}}{\pgfqpoint{0.1cm}{1.258cm}}{\pgfqpoint{0.137cm}{1.258cm}}
\pgfpathcurveto{\pgfqpoint{0.173cm}{1.258cm}}{\pgfqpoint{0.207cm}{1.273cm}}{\pgfqpoint{0.233cm}{1.299cm}}
\pgfpathcurveto{\pgfqpoint{0.259cm}{1.324cm}}{\pgfqpoint{0.273cm}{1.359cm}}{\pgfqpoint{0.273cm}{1.395cm}}
\pgfusepath{fill}
\begin{pgfscope}
\pgfsetdash{}{0cm}
\pgfsetlinewidth{0.818mm}
\pgfsetmiterlimit{7.0}
\pgfpathmoveto{\pgfqpoint{0.682cm}{0.671cm}}
\pgfpathlineto{\pgfqpoint{0.679cm}{1.418cm}}
\pgfusepath{stroke}
\end{pgfscope}
\pgfpathmoveto{\pgfqpoint{0.815cm}{1.399cm}}
\pgfpathcurveto{\pgfqpoint{0.815cm}{1.435cm}}{\pgfqpoint{0.801cm}{1.47cm}}{\pgfqpoint{0.775cm}{1.496cm}}
\pgfpathcurveto{\pgfqpoint{0.75cm}{1.521cm}}{\pgfqpoint{0.715cm}{1.536cm}}{\pgfqpoint{0.679cm}{1.536cm}}
\pgfpathcurveto{\pgfqpoint{0.643cm}{1.536cm}}{\pgfqpoint{0.608cm}{1.521cm}}{\pgfqpoint{0.582cm}{1.496cm}}
\pgfpathcurveto{\pgfqpoint{0.557cm}{1.47cm}}{\pgfqpoint{0.542cm}{1.435cm}}{\pgfqpoint{0.542cm}{1.399cm}}
\pgfpathcurveto{\pgfqpoint{0.542cm}{1.363cm}}{\pgfqpoint{0.557cm}{1.328cm}}{\pgfqpoint{0.582cm}{1.302cm}}
\pgfpathcurveto{\pgfqpoint{0.608cm}{1.276cm}}{\pgfqpoint{0.643cm}{1.262cm}}{\pgfqpoint{0.679cm}{1.262cm}}
\pgfpathcurveto{\pgfqpoint{0.715cm}{1.262cm}}{\pgfqpoint{0.75cm}{1.276cm}}{\pgfqpoint{0.775cm}{1.302cm}}
\pgfpathcurveto{\pgfqpoint{0.801cm}{1.328cm}}{\pgfqpoint{0.815cm}{1.363cm}}{\pgfqpoint{0.815cm}{1.399cm}}
\pgfusepath{fill}
\pgfpathmoveto{\pgfqpoint{1.345cm}{1.371cm}}
\pgfpathcurveto{\pgfqpoint{1.345cm}{1.408cm}}{\pgfqpoint{1.331cm}{1.442cm}}{\pgfqpoint{1.305cm}{1.468cm}}
\pgfpathcurveto{\pgfqpoint{1.28cm}{1.494cm}}{\pgfqpoint{1.245cm}{1.508cm}}{\pgfqpoint{1.209cm}{1.508cm}}
\pgfpathcurveto{\pgfqpoint{1.172cm}{1.508cm}}{\pgfqpoint{1.138cm}{1.494cm}}{\pgfqpoint{1.112cm}{1.468cm}}
\pgfpathcurveto{\pgfqpoint{1.087cm}{1.442cm}}{\pgfqpoint{1.072cm}{1.408cm}}{\pgfqpoint{1.072cm}{1.371cm}}
\pgfpathcurveto{\pgfqpoint{1.072cm}{1.335cm}}{\pgfqpoint{1.087cm}{1.3cm}}{\pgfqpoint{1.112cm}{1.274cm}}
\pgfpathcurveto{\pgfqpoint{1.138cm}{1.249cm}}{\pgfqpoint{1.172cm}{1.234cm}}{\pgfqpoint{1.209cm}{1.234cm}}
\pgfpathcurveto{\pgfqpoint{1.245cm}{1.234cm}}{\pgfqpoint{1.28cm}{1.249cm}}{\pgfqpoint{1.305cm}{1.274cm}}
\pgfpathcurveto{\pgfqpoint{1.331cm}{1.3cm}}{\pgfqpoint{1.345cm}{1.335cm}}{\pgfqpoint{1.345cm}{1.371cm}}
\pgfusepath{fill}
\begin{pgfscope}
\pgfsetdash{}{0cm}
\pgfsetlinewidth{0.818mm}
\pgfsetroundcap
\pgfsetmiterlimit{4.0}
\pgfpathmoveto{\pgfqpoint{0.682cm}{0.671cm}}
\pgfpathlineto{\pgfqpoint{0.682cm}{0.042cm}}
\pgfusepath{stroke}
\end{pgfscope}
\end{pgfscope}
\end{pgfscope}
\end{pgfscope}
\end{tikzpicture}}} \|_{C_T^{\beta /
     2} L^{\infty, \varepsilon} (\rho^{\sigma})} +  \delta \| Y_{M,
     \varepsilon} \|_{C_T \CC^{1 / 2 - \kappa, \varepsilon} (\rho^{\sigma})}
  \]
  \[ \lesssim \lambda \| X_{M, \varepsilon}^{\!\resizebox{0.6em}{!}{
\begin{tikzpicture}
\pgfpathmoveto{\pgfqpoint{0cm}{-0.035cm}}
\pgfpathlineto{\pgfqpoint{1.376cm}{-0.035cm}}
\pgfpathlineto{\pgfqpoint{1.376cm}{1.552cm}}
\pgfpathlineto{\pgfqpoint{0cm}{1.552cm}}
\pgfpathclose
\pgfusepath{clip}
\begin{pgfscope}
\begin{pgfscope}
\pgfpathmoveto{\pgfqpoint{0cm}{-0.035cm}}
\pgfpathlineto{\pgfqpoint{1.376cm}{-0.035cm}}
\pgfpathlineto{\pgfqpoint{1.376cm}{1.552cm}}
\pgfpathlineto{\pgfqpoint{0cm}{1.552cm}}
\pgfpathclose
\pgfusepath{clip}
\begin{pgfscope}
\begin{pgfscope}
\pgfsetdash{}{0cm}
\pgfsetlinewidth{0.818mm}
\pgfsetroundcap
\pgfsetroundjoin
\pgfsetmiterlimit{7.0}
\definecolor{eps2pgf_color}{gray}{0}\pgfsetstrokecolor{eps2pgf_color}\pgfsetfillcolor{eps2pgf_color}
\pgfpathmoveto{\pgfqpoint{0.117cm}{1.421cm}}
\pgfpathlineto{\pgfqpoint{0.682cm}{0.671cm}}
\pgfpathlineto{\pgfqpoint{1.246cm}{1.421cm}}
\pgfusepath{stroke}
\end{pgfscope}
\definecolor{eps2pgf_color}{gray}{0}\pgfsetstrokecolor{eps2pgf_color}\pgfsetfillcolor{eps2pgf_color}
\pgfpathmoveto{\pgfqpoint{0.273cm}{1.395cm}}
\pgfpathcurveto{\pgfqpoint{0.273cm}{1.432cm}}{\pgfqpoint{0.259cm}{1.467cm}}{\pgfqpoint{0.233cm}{1.492cm}}
\pgfpathcurveto{\pgfqpoint{0.207cm}{1.518cm}}{\pgfqpoint{0.173cm}{1.532cm}}{\pgfqpoint{0.137cm}{1.532cm}}
\pgfpathcurveto{\pgfqpoint{0.1cm}{1.532cm}}{\pgfqpoint{0.066cm}{1.518cm}}{\pgfqpoint{0.04cm}{1.492cm}}
\pgfpathcurveto{\pgfqpoint{0.014cm}{1.467cm}}{\pgfqpoint{0cm}{1.432cm}}{\pgfqpoint{0cm}{1.395cm}}
\pgfpathcurveto{\pgfqpoint{0cm}{1.359cm}}{\pgfqpoint{0.014cm}{1.324cm}}{\pgfqpoint{0.04cm}{1.299cm}}
\pgfpathcurveto{\pgfqpoint{0.066cm}{1.273cm}}{\pgfqpoint{0.1cm}{1.258cm}}{\pgfqpoint{0.137cm}{1.258cm}}
\pgfpathcurveto{\pgfqpoint{0.173cm}{1.258cm}}{\pgfqpoint{0.207cm}{1.273cm}}{\pgfqpoint{0.233cm}{1.299cm}}
\pgfpathcurveto{\pgfqpoint{0.259cm}{1.324cm}}{\pgfqpoint{0.273cm}{1.359cm}}{\pgfqpoint{0.273cm}{1.395cm}}
\pgfusepath{fill}
\begin{pgfscope}
\pgfsetdash{}{0cm}
\pgfsetlinewidth{0.818mm}
\pgfsetmiterlimit{7.0}
\pgfpathmoveto{\pgfqpoint{0.682cm}{0.671cm}}
\pgfpathlineto{\pgfqpoint{0.679cm}{1.418cm}}
\pgfusepath{stroke}
\end{pgfscope}
\pgfpathmoveto{\pgfqpoint{0.815cm}{1.399cm}}
\pgfpathcurveto{\pgfqpoint{0.815cm}{1.435cm}}{\pgfqpoint{0.801cm}{1.47cm}}{\pgfqpoint{0.775cm}{1.496cm}}
\pgfpathcurveto{\pgfqpoint{0.75cm}{1.521cm}}{\pgfqpoint{0.715cm}{1.536cm}}{\pgfqpoint{0.679cm}{1.536cm}}
\pgfpathcurveto{\pgfqpoint{0.643cm}{1.536cm}}{\pgfqpoint{0.608cm}{1.521cm}}{\pgfqpoint{0.582cm}{1.496cm}}
\pgfpathcurveto{\pgfqpoint{0.557cm}{1.47cm}}{\pgfqpoint{0.542cm}{1.435cm}}{\pgfqpoint{0.542cm}{1.399cm}}
\pgfpathcurveto{\pgfqpoint{0.542cm}{1.363cm}}{\pgfqpoint{0.557cm}{1.328cm}}{\pgfqpoint{0.582cm}{1.302cm}}
\pgfpathcurveto{\pgfqpoint{0.608cm}{1.276cm}}{\pgfqpoint{0.643cm}{1.262cm}}{\pgfqpoint{0.679cm}{1.262cm}}
\pgfpathcurveto{\pgfqpoint{0.715cm}{1.262cm}}{\pgfqpoint{0.75cm}{1.276cm}}{\pgfqpoint{0.775cm}{1.302cm}}
\pgfpathcurveto{\pgfqpoint{0.801cm}{1.328cm}}{\pgfqpoint{0.815cm}{1.363cm}}{\pgfqpoint{0.815cm}{1.399cm}}
\pgfusepath{fill}
\pgfpathmoveto{\pgfqpoint{1.345cm}{1.371cm}}
\pgfpathcurveto{\pgfqpoint{1.345cm}{1.408cm}}{\pgfqpoint{1.331cm}{1.442cm}}{\pgfqpoint{1.305cm}{1.468cm}}
\pgfpathcurveto{\pgfqpoint{1.28cm}{1.494cm}}{\pgfqpoint{1.245cm}{1.508cm}}{\pgfqpoint{1.209cm}{1.508cm}}
\pgfpathcurveto{\pgfqpoint{1.172cm}{1.508cm}}{\pgfqpoint{1.138cm}{1.494cm}}{\pgfqpoint{1.112cm}{1.468cm}}
\pgfpathcurveto{\pgfqpoint{1.087cm}{1.442cm}}{\pgfqpoint{1.072cm}{1.408cm}}{\pgfqpoint{1.072cm}{1.371cm}}
\pgfpathcurveto{\pgfqpoint{1.072cm}{1.335cm}}{\pgfqpoint{1.087cm}{1.3cm}}{\pgfqpoint{1.112cm}{1.274cm}}
\pgfpathcurveto{\pgfqpoint{1.138cm}{1.249cm}}{\pgfqpoint{1.172cm}{1.234cm}}{\pgfqpoint{1.209cm}{1.234cm}}
\pgfpathcurveto{\pgfqpoint{1.245cm}{1.234cm}}{\pgfqpoint{1.28cm}{1.249cm}}{\pgfqpoint{1.305cm}{1.274cm}}
\pgfpathcurveto{\pgfqpoint{1.331cm}{1.3cm}}{\pgfqpoint{1.345cm}{1.335cm}}{\pgfqpoint{1.345cm}{1.371cm}}
\pgfusepath{fill}
\begin{pgfscope}
\pgfsetdash{}{0cm}
\pgfsetlinewidth{0.818mm}
\pgfsetroundcap
\pgfsetmiterlimit{4.0}
\pgfpathmoveto{\pgfqpoint{0.682cm}{0.671cm}}
\pgfpathlineto{\pgfqpoint{0.682cm}{0.042cm}}
\pgfusepath{stroke}
\end{pgfscope}
\end{pgfscope}
\end{pgfscope}
\end{pgfscope}
\end{tikzpicture}}} \|_{C_T^{\beta / 2}
     L^{\infty, \varepsilon} (\rho^{\sigma})} + \lambda \| X_{M,
     \varepsilon}^{\!\resizebox{0.6em}{!}{
\begin{tikzpicture}
\pgfpathmoveto{\pgfqpoint{0cm}{-0.035cm}}
\pgfpathlineto{\pgfqpoint{1.376cm}{-0.035cm}}
\pgfpathlineto{\pgfqpoint{1.376cm}{1.552cm}}
\pgfpathlineto{\pgfqpoint{0cm}{1.552cm}}
\pgfpathclose
\pgfusepath{clip}
\begin{pgfscope}
\begin{pgfscope}
\pgfpathmoveto{\pgfqpoint{0cm}{-0.035cm}}
\pgfpathlineto{\pgfqpoint{1.376cm}{-0.035cm}}
\pgfpathlineto{\pgfqpoint{1.376cm}{1.552cm}}
\pgfpathlineto{\pgfqpoint{0cm}{1.552cm}}
\pgfpathclose
\pgfusepath{clip}
\begin{pgfscope}
\begin{pgfscope}
\pgfsetdash{}{0cm}
\pgfsetlinewidth{0.818mm}
\pgfsetroundcap
\pgfsetroundjoin
\pgfsetmiterlimit{7.0}
\definecolor{eps2pgf_color}{gray}{0}\pgfsetstrokecolor{eps2pgf_color}\pgfsetfillcolor{eps2pgf_color}
\pgfpathmoveto{\pgfqpoint{0.117cm}{1.421cm}}
\pgfpathlineto{\pgfqpoint{0.682cm}{0.671cm}}
\pgfpathlineto{\pgfqpoint{1.246cm}{1.421cm}}
\pgfusepath{stroke}
\end{pgfscope}
\definecolor{eps2pgf_color}{gray}{0}\pgfsetstrokecolor{eps2pgf_color}\pgfsetfillcolor{eps2pgf_color}
\pgfpathmoveto{\pgfqpoint{0.273cm}{1.395cm}}
\pgfpathcurveto{\pgfqpoint{0.273cm}{1.432cm}}{\pgfqpoint{0.259cm}{1.467cm}}{\pgfqpoint{0.233cm}{1.492cm}}
\pgfpathcurveto{\pgfqpoint{0.207cm}{1.518cm}}{\pgfqpoint{0.173cm}{1.532cm}}{\pgfqpoint{0.137cm}{1.532cm}}
\pgfpathcurveto{\pgfqpoint{0.1cm}{1.532cm}}{\pgfqpoint{0.066cm}{1.518cm}}{\pgfqpoint{0.04cm}{1.492cm}}
\pgfpathcurveto{\pgfqpoint{0.014cm}{1.467cm}}{\pgfqpoint{0cm}{1.432cm}}{\pgfqpoint{0cm}{1.395cm}}
\pgfpathcurveto{\pgfqpoint{0cm}{1.359cm}}{\pgfqpoint{0.014cm}{1.324cm}}{\pgfqpoint{0.04cm}{1.299cm}}
\pgfpathcurveto{\pgfqpoint{0.066cm}{1.273cm}}{\pgfqpoint{0.1cm}{1.258cm}}{\pgfqpoint{0.137cm}{1.258cm}}
\pgfpathcurveto{\pgfqpoint{0.173cm}{1.258cm}}{\pgfqpoint{0.207cm}{1.273cm}}{\pgfqpoint{0.233cm}{1.299cm}}
\pgfpathcurveto{\pgfqpoint{0.259cm}{1.324cm}}{\pgfqpoint{0.273cm}{1.359cm}}{\pgfqpoint{0.273cm}{1.395cm}}
\pgfusepath{fill}
\begin{pgfscope}
\pgfsetdash{}{0cm}
\pgfsetlinewidth{0.818mm}
\pgfsetmiterlimit{7.0}
\pgfpathmoveto{\pgfqpoint{0.682cm}{0.671cm}}
\pgfpathlineto{\pgfqpoint{0.679cm}{1.418cm}}
\pgfusepath{stroke}
\end{pgfscope}
\pgfpathmoveto{\pgfqpoint{0.815cm}{1.399cm}}
\pgfpathcurveto{\pgfqpoint{0.815cm}{1.435cm}}{\pgfqpoint{0.801cm}{1.47cm}}{\pgfqpoint{0.775cm}{1.496cm}}
\pgfpathcurveto{\pgfqpoint{0.75cm}{1.521cm}}{\pgfqpoint{0.715cm}{1.536cm}}{\pgfqpoint{0.679cm}{1.536cm}}
\pgfpathcurveto{\pgfqpoint{0.643cm}{1.536cm}}{\pgfqpoint{0.608cm}{1.521cm}}{\pgfqpoint{0.582cm}{1.496cm}}
\pgfpathcurveto{\pgfqpoint{0.557cm}{1.47cm}}{\pgfqpoint{0.542cm}{1.435cm}}{\pgfqpoint{0.542cm}{1.399cm}}
\pgfpathcurveto{\pgfqpoint{0.542cm}{1.363cm}}{\pgfqpoint{0.557cm}{1.328cm}}{\pgfqpoint{0.582cm}{1.302cm}}
\pgfpathcurveto{\pgfqpoint{0.608cm}{1.276cm}}{\pgfqpoint{0.643cm}{1.262cm}}{\pgfqpoint{0.679cm}{1.262cm}}
\pgfpathcurveto{\pgfqpoint{0.715cm}{1.262cm}}{\pgfqpoint{0.75cm}{1.276cm}}{\pgfqpoint{0.775cm}{1.302cm}}
\pgfpathcurveto{\pgfqpoint{0.801cm}{1.328cm}}{\pgfqpoint{0.815cm}{1.363cm}}{\pgfqpoint{0.815cm}{1.399cm}}
\pgfusepath{fill}
\pgfpathmoveto{\pgfqpoint{1.345cm}{1.371cm}}
\pgfpathcurveto{\pgfqpoint{1.345cm}{1.408cm}}{\pgfqpoint{1.331cm}{1.442cm}}{\pgfqpoint{1.305cm}{1.468cm}}
\pgfpathcurveto{\pgfqpoint{1.28cm}{1.494cm}}{\pgfqpoint{1.245cm}{1.508cm}}{\pgfqpoint{1.209cm}{1.508cm}}
\pgfpathcurveto{\pgfqpoint{1.172cm}{1.508cm}}{\pgfqpoint{1.138cm}{1.494cm}}{\pgfqpoint{1.112cm}{1.468cm}}
\pgfpathcurveto{\pgfqpoint{1.087cm}{1.442cm}}{\pgfqpoint{1.072cm}{1.408cm}}{\pgfqpoint{1.072cm}{1.371cm}}
\pgfpathcurveto{\pgfqpoint{1.072cm}{1.335cm}}{\pgfqpoint{1.087cm}{1.3cm}}{\pgfqpoint{1.112cm}{1.274cm}}
\pgfpathcurveto{\pgfqpoint{1.138cm}{1.249cm}}{\pgfqpoint{1.172cm}{1.234cm}}{\pgfqpoint{1.209cm}{1.234cm}}
\pgfpathcurveto{\pgfqpoint{1.245cm}{1.234cm}}{\pgfqpoint{1.28cm}{1.249cm}}{\pgfqpoint{1.305cm}{1.274cm}}
\pgfpathcurveto{\pgfqpoint{1.331cm}{1.3cm}}{\pgfqpoint{1.345cm}{1.335cm}}{\pgfqpoint{1.345cm}{1.371cm}}
\pgfusepath{fill}
\begin{pgfscope}
\pgfsetdash{}{0cm}
\pgfsetlinewidth{0.818mm}
\pgfsetroundcap
\pgfsetmiterlimit{4.0}
\pgfpathmoveto{\pgfqpoint{0.682cm}{0.671cm}}
\pgfpathlineto{\pgfqpoint{0.682cm}{0.042cm}}
\pgfusepath{stroke}
\end{pgfscope}
\end{pgfscope}
\end{pgfscope}
\end{pgfscope}
\end{tikzpicture}}} \|_{C_T \CC^{1 / 2 - \kappa, \varepsilon}
     (\rho^{\sigma})} . \]
  The proof is complete.
\end{proof}

According to this result, we remark that $Y_{M, \varepsilon}$ itself is not a
polynomial in the noise terms, but with our choice of localization it allows
for a polynomial bound of its norm.
As the next step, we introduce further stochastic objects needed below.
Namely,
\[ X_{M, \varepsilon}^{\!\resizebox{0.6em}{!}{
\begin{tikzpicture}
\pgfpathmoveto{\pgfqpoint{0cm}{0cm}}
\pgfpathlineto{\pgfqpoint{1.376cm}{0cm}}
\pgfpathlineto{\pgfqpoint{1.376cm}{1.588cm}}
\pgfpathlineto{\pgfqpoint{0cm}{1.588cm}}
\pgfpathclose
\pgfusepath{clip}
\begin{pgfscope}
\begin{pgfscope}
\pgfpathmoveto{\pgfqpoint{0cm}{0cm}}
\pgfpathlineto{\pgfqpoint{1.376cm}{0cm}}
\pgfpathlineto{\pgfqpoint{1.376cm}{1.588cm}}
\pgfpathlineto{\pgfqpoint{0cm}{1.588cm}}
\pgfpathclose
\pgfusepath{clip}
\begin{pgfscope}
\begin{pgfscope}
\definecolor{eps2pgf_color}{gray}{0.976471}\pgfsetstrokecolor{eps2pgf_color}\pgfsetfillcolor{eps2pgf_color}
\pgfpathmoveto{\pgfqpoint{0cm}{0cm}}
\pgfpathlineto{\pgfqpoint{1.376cm}{0cm}}
\pgfpathlineto{\pgfqpoint{1.376cm}{1.588cm}}
\pgfpathlineto{\pgfqpoint{0cm}{1.588cm}}
\pgfpathclose
\pgfusepath{fill}
\end{pgfscope}
\begin{pgfscope}
\pgfsetdash{}{0cm}
\pgfsetlinewidth{0.818mm}
\pgfsetroundcap
\pgfsetroundjoin
\pgfsetmiterlimit{7.0}
\definecolor{eps2pgf_color}{gray}{0}\pgfsetstrokecolor{eps2pgf_color}\pgfsetfillcolor{eps2pgf_color}
\pgfpathmoveto{\pgfqpoint{0.117cm}{1.476cm}}
\pgfpathlineto{\pgfqpoint{0.682cm}{0.726cm}}
\pgfpathlineto{\pgfqpoint{1.246cm}{1.476cm}}
\pgfusepath{stroke}
\end{pgfscope}
\definecolor{eps2pgf_color}{gray}{0}\pgfsetstrokecolor{eps2pgf_color}\pgfsetfillcolor{eps2pgf_color}
\pgfpathmoveto{\pgfqpoint{0.273cm}{1.451cm}}
\pgfpathcurveto{\pgfqpoint{0.273cm}{1.487cm}}{\pgfqpoint{0.259cm}{1.522cm}}{\pgfqpoint{0.233cm}{1.547cm}}
\pgfpathcurveto{\pgfqpoint{0.207cm}{1.573cm}}{\pgfqpoint{0.173cm}{1.588cm}}{\pgfqpoint{0.137cm}{1.588cm}}
\pgfpathcurveto{\pgfqpoint{0.1cm}{1.588cm}}{\pgfqpoint{0.066cm}{1.573cm}}{\pgfqpoint{0.04cm}{1.547cm}}
\pgfpathcurveto{\pgfqpoint{0.014cm}{1.522cm}}{\pgfqpoint{0cm}{1.487cm}}{\pgfqpoint{0cm}{1.451cm}}
\pgfpathcurveto{\pgfqpoint{0cm}{1.414cm}}{\pgfqpoint{0.014cm}{1.379cm}}{\pgfqpoint{0.04cm}{1.354cm}}
\pgfpathcurveto{\pgfqpoint{0.066cm}{1.328cm}}{\pgfqpoint{0.1cm}{1.314cm}}{\pgfqpoint{0.137cm}{1.314cm}}
\pgfpathcurveto{\pgfqpoint{0.173cm}{1.314cm}}{\pgfqpoint{0.207cm}{1.328cm}}{\pgfqpoint{0.233cm}{1.354cm}}
\pgfpathcurveto{\pgfqpoint{0.259cm}{1.379cm}}{\pgfqpoint{0.273cm}{1.414cm}}{\pgfqpoint{0.273cm}{1.451cm}}
\pgfusepath{fill}
\pgfpathmoveto{\pgfqpoint{1.345cm}{1.426cm}}
\pgfpathcurveto{\pgfqpoint{1.345cm}{1.463cm}}{\pgfqpoint{1.331cm}{1.497cm}}{\pgfqpoint{1.305cm}{1.523cm}}
\pgfpathcurveto{\pgfqpoint{1.28cm}{1.549cm}}{\pgfqpoint{1.245cm}{1.563cm}}{\pgfqpoint{1.209cm}{1.563cm}}
\pgfpathcurveto{\pgfqpoint{1.172cm}{1.563cm}}{\pgfqpoint{1.138cm}{1.549cm}}{\pgfqpoint{1.112cm}{1.523cm}}
\pgfpathcurveto{\pgfqpoint{1.087cm}{1.497cm}}{\pgfqpoint{1.072cm}{1.463cm}}{\pgfqpoint{1.072cm}{1.426cm}}
\pgfpathcurveto{\pgfqpoint{1.072cm}{1.39cm}}{\pgfqpoint{1.087cm}{1.355cm}}{\pgfqpoint{1.112cm}{1.329cm}}
\pgfpathcurveto{\pgfqpoint{1.138cm}{1.304cm}}{\pgfqpoint{1.172cm}{1.289cm}}{\pgfqpoint{1.209cm}{1.289cm}}
\pgfpathcurveto{\pgfqpoint{1.245cm}{1.289cm}}{\pgfqpoint{1.28cm}{1.304cm}}{\pgfqpoint{1.305cm}{1.329cm}}
\pgfpathcurveto{\pgfqpoint{1.331cm}{1.355cm}}{\pgfqpoint{1.345cm}{1.39cm}}{\pgfqpoint{1.345cm}{1.426cm}}
\pgfusepath{fill}
\begin{pgfscope}
\pgfsetdash{}{0cm}
\pgfsetlinewidth{0.818mm}
\pgfsetroundcap
\pgfsetmiterlimit{4.0}
\pgfpathmoveto{\pgfqpoint{0.682cm}{0.726cm}}
\pgfpathlineto{\pgfqpoint{0.682cm}{0.097cm}}
\pgfusepath{stroke}
\end{pgfscope}
\end{pgfscope}
\end{pgfscope}
\end{pgfscope}
\end{tikzpicture}}} \assign \LL_{\varepsilon}^{- 1} \llbracket
   X_{M, \varepsilon}^2 \rrbracket, \qquad X_{M, \varepsilon}^{\!\resizebox{!}{.8em}{
\begin{tikzpicture}
\pgfpathmoveto{\pgfqpoint{0cm}{-0.035cm}}
\pgfpathlineto{\pgfqpoint{1.976cm}{-0.035cm}}
\pgfpathlineto{\pgfqpoint{1.976cm}{1.94cm}}
\pgfpathlineto{\pgfqpoint{0cm}{1.94cm}}
\pgfpathclose
\pgfusepath{clip}
\begin{pgfscope}
\begin{pgfscope}
\pgfpathmoveto{\pgfqpoint{0cm}{-0.035cm}}
\pgfpathlineto{\pgfqpoint{1.976cm}{-0.035cm}}
\pgfpathlineto{\pgfqpoint{1.976cm}{1.94cm}}
\pgfpathlineto{\pgfqpoint{0cm}{1.94cm}}
\pgfpathclose
\pgfusepath{clip}
\begin{pgfscope}
\begin{pgfscope}
\pgfsetdash{}{0cm}
\pgfsetlinewidth{0.818mm}
\pgfsetroundcap
\pgfsetroundjoin
\pgfsetmiterlimit{7.0}
\definecolor{eps2pgf_color}{gray}{0}\pgfsetstrokecolor{eps2pgf_color}\pgfsetfillcolor{eps2pgf_color}
\pgfpathmoveto{\pgfqpoint{0.117cm}{1.815cm}}
\pgfpathlineto{\pgfqpoint{0.682cm}{1.065cm}}
\pgfpathlineto{\pgfqpoint{1.246cm}{1.815cm}}
\pgfusepath{stroke}
\end{pgfscope}
\definecolor{eps2pgf_color}{gray}{0}\pgfsetstrokecolor{eps2pgf_color}\pgfsetfillcolor{eps2pgf_color}
\pgfpathmoveto{\pgfqpoint{0.273cm}{1.789cm}}
\pgfpathcurveto{\pgfqpoint{0.273cm}{1.825cm}}{\pgfqpoint{0.259cm}{1.86cm}}{\pgfqpoint{0.233cm}{1.886cm}}
\pgfpathcurveto{\pgfqpoint{0.207cm}{1.912cm}}{\pgfqpoint{0.173cm}{1.926cm}}{\pgfqpoint{0.137cm}{1.926cm}}
\pgfpathcurveto{\pgfqpoint{0.1cm}{1.926cm}}{\pgfqpoint{0.066cm}{1.912cm}}{\pgfqpoint{0.04cm}{1.886cm}}
\pgfpathcurveto{\pgfqpoint{0.014cm}{1.86cm}}{\pgfqpoint{0cm}{1.825cm}}{\pgfqpoint{0cm}{1.789cm}}
\pgfpathcurveto{\pgfqpoint{0cm}{1.753cm}}{\pgfqpoint{0.014cm}{1.718cm}}{\pgfqpoint{0.04cm}{1.692cm}}
\pgfpathcurveto{\pgfqpoint{0.066cm}{1.667cm}}{\pgfqpoint{0.1cm}{1.652cm}}{\pgfqpoint{0.137cm}{1.652cm}}
\pgfpathcurveto{\pgfqpoint{0.173cm}{1.652cm}}{\pgfqpoint{0.207cm}{1.667cm}}{\pgfqpoint{0.233cm}{1.692cm}}
\pgfpathcurveto{\pgfqpoint{0.259cm}{1.718cm}}{\pgfqpoint{0.273cm}{1.753cm}}{\pgfqpoint{0.273cm}{1.789cm}}
\pgfusepath{fill}
\begin{pgfscope}
\pgfsetdash{}{0cm}
\pgfsetlinewidth{0.818mm}
\pgfsetmiterlimit{7.0}
\pgfpathmoveto{\pgfqpoint{0.682cm}{1.065cm}}
\pgfpathlineto{\pgfqpoint{0.679cm}{1.812cm}}
\pgfusepath{stroke}
\end{pgfscope}
\pgfpathmoveto{\pgfqpoint{0.815cm}{1.793cm}}
\pgfpathcurveto{\pgfqpoint{0.815cm}{1.829cm}}{\pgfqpoint{0.801cm}{1.864cm}}{\pgfqpoint{0.775cm}{1.89cm}}
\pgfpathcurveto{\pgfqpoint{0.75cm}{1.915cm}}{\pgfqpoint{0.715cm}{1.93cm}}{\pgfqpoint{0.679cm}{1.93cm}}
\pgfpathcurveto{\pgfqpoint{0.643cm}{1.93cm}}{\pgfqpoint{0.608cm}{1.915cm}}{\pgfqpoint{0.582cm}{1.89cm}}
\pgfpathcurveto{\pgfqpoint{0.557cm}{1.864cm}}{\pgfqpoint{0.542cm}{1.829cm}}{\pgfqpoint{0.542cm}{1.793cm}}
\pgfpathcurveto{\pgfqpoint{0.542cm}{1.756cm}}{\pgfqpoint{0.557cm}{1.722cm}}{\pgfqpoint{0.582cm}{1.696cm}}
\pgfpathcurveto{\pgfqpoint{0.608cm}{1.67cm}}{\pgfqpoint{0.643cm}{1.656cm}}{\pgfqpoint{0.679cm}{1.656cm}}
\pgfpathcurveto{\pgfqpoint{0.715cm}{1.656cm}}{\pgfqpoint{0.75cm}{1.67cm}}{\pgfqpoint{0.775cm}{1.696cm}}
\pgfpathcurveto{\pgfqpoint{0.801cm}{1.722cm}}{\pgfqpoint{0.815cm}{1.756cm}}{\pgfqpoint{0.815cm}{1.793cm}}
\pgfusepath{fill}
\pgfpathmoveto{\pgfqpoint{1.345cm}{1.765cm}}
\pgfpathcurveto{\pgfqpoint{1.345cm}{1.801cm}}{\pgfqpoint{1.331cm}{1.836cm}}{\pgfqpoint{1.305cm}{1.862cm}}
\pgfpathcurveto{\pgfqpoint{1.28cm}{1.887cm}}{\pgfqpoint{1.245cm}{1.902cm}}{\pgfqpoint{1.209cm}{1.902cm}}
\pgfpathcurveto{\pgfqpoint{1.172cm}{1.902cm}}{\pgfqpoint{1.138cm}{1.887cm}}{\pgfqpoint{1.112cm}{1.862cm}}
\pgfpathcurveto{\pgfqpoint{1.087cm}{1.836cm}}{\pgfqpoint{1.072cm}{1.801cm}}{\pgfqpoint{1.072cm}{1.765cm}}
\pgfpathcurveto{\pgfqpoint{1.072cm}{1.728cm}}{\pgfqpoint{1.087cm}{1.694cm}}{\pgfqpoint{1.112cm}{1.668cm}}
\pgfpathcurveto{\pgfqpoint{1.138cm}{1.642cm}}{\pgfqpoint{1.172cm}{1.628cm}}{\pgfqpoint{1.209cm}{1.628cm}}
\pgfpathcurveto{\pgfqpoint{1.245cm}{1.628cm}}{\pgfqpoint{1.28cm}{1.642cm}}{\pgfqpoint{1.305cm}{1.668cm}}
\pgfpathcurveto{\pgfqpoint{1.331cm}{1.694cm}}{\pgfqpoint{1.345cm}{1.728cm}}{\pgfqpoint{1.345cm}{1.765cm}}
\pgfusepath{fill}
\begin{pgfscope}
\pgfsetdash{}{0cm}
\pgfsetlinewidth{0.818mm}
\pgfsetroundcap
\pgfsetroundjoin
\pgfsetmiterlimit{7.0}
\pgfpathmoveto{\pgfqpoint{0.682cm}{1.065cm}}
\pgfpathlineto{\pgfqpoint{1.246cm}{0.315cm}}
\pgfpathlineto{\pgfqpoint{1.811cm}{1.065cm}}
\pgfusepath{stroke}
\end{pgfscope}
\pgfpathmoveto{\pgfqpoint{1.948cm}{1.065cm}}
\pgfpathcurveto{\pgfqpoint{1.948cm}{1.101cm}}{\pgfqpoint{1.933cm}{1.136cm}}{\pgfqpoint{1.907cm}{1.162cm}}
\pgfpathcurveto{\pgfqpoint{1.882cm}{1.187cm}}{\pgfqpoint{1.847cm}{1.202cm}}{\pgfqpoint{1.811cm}{1.202cm}}
\pgfpathcurveto{\pgfqpoint{1.775cm}{1.202cm}}{\pgfqpoint{1.74cm}{1.187cm}}{\pgfqpoint{1.714cm}{1.162cm}}
\pgfpathcurveto{\pgfqpoint{1.689cm}{1.136cm}}{\pgfqpoint{1.674cm}{1.101cm}}{\pgfqpoint{1.674cm}{1.065cm}}
\pgfpathcurveto{\pgfqpoint{1.674cm}{1.029cm}}{\pgfqpoint{1.689cm}{0.994cm}}{\pgfqpoint{1.714cm}{0.968cm}}
\pgfpathcurveto{\pgfqpoint{1.74cm}{0.942cm}}{\pgfqpoint{1.775cm}{0.928cm}}{\pgfqpoint{1.811cm}{0.928cm}}
\pgfpathcurveto{\pgfqpoint{1.847cm}{0.928cm}}{\pgfqpoint{1.882cm}{0.942cm}}{\pgfqpoint{1.907cm}{0.968cm}}
\pgfpathcurveto{\pgfqpoint{1.933cm}{0.994cm}}{\pgfqpoint{1.948cm}{1.029cm}}{\pgfqpoint{1.948cm}{1.065cm}}
\pgfusepath{fill}
\begin{pgfscope}
\pgfsetdash{}{0cm}
\pgfsetlinewidth{0.818mm}
\pgfsetmiterlimit{4.0}
\pgfpathmoveto{\pgfqpoint{1.383cm}{0.178cm}}
\pgfpathcurveto{\pgfqpoint{1.383cm}{0.214cm}}{\pgfqpoint{1.369cm}{0.249cm}}{\pgfqpoint{1.343cm}{0.275cm}}
\pgfpathcurveto{\pgfqpoint{1.317cm}{0.3cm}}{\pgfqpoint{1.283cm}{0.315cm}}{\pgfqpoint{1.246cm}{0.315cm}}
\pgfpathcurveto{\pgfqpoint{1.21cm}{0.315cm}}{\pgfqpoint{1.175cm}{0.3cm}}{\pgfqpoint{1.15cm}{0.275cm}}
\pgfpathcurveto{\pgfqpoint{1.124cm}{0.249cm}}{\pgfqpoint{1.11cm}{0.214cm}}{\pgfqpoint{1.11cm}{0.178cm}}
\pgfpathcurveto{\pgfqpoint{1.11cm}{0.141cm}}{\pgfqpoint{1.124cm}{0.107cm}}{\pgfqpoint{1.15cm}{0.081cm}}
\pgfpathcurveto{\pgfqpoint{1.175cm}{0.055cm}}{\pgfqpoint{1.21cm}{0.041cm}}{\pgfqpoint{1.246cm}{0.041cm}}
\pgfpathcurveto{\pgfqpoint{1.283cm}{0.041cm}}{\pgfqpoint{1.317cm}{0.055cm}}{\pgfqpoint{1.343cm}{0.081cm}}
\pgfpathcurveto{\pgfqpoint{1.369cm}{0.107cm}}{\pgfqpoint{1.383cm}{0.141cm}}{\pgfqpoint{1.383cm}{0.178cm}}
\pgfusepath{stroke}
\end{pgfscope}
\end{pgfscope}
\end{pgfscope}
\end{pgfscope}
\end{tikzpicture}}} =
   X_{M, \varepsilon} \circ X_{M, \varepsilon}^{\!\resizebox{0.6em}{!}{
\begin{tikzpicture}
\pgfpathmoveto{\pgfqpoint{0cm}{-0.035cm}}
\pgfpathlineto{\pgfqpoint{1.376cm}{-0.035cm}}
\pgfpathlineto{\pgfqpoint{1.376cm}{1.552cm}}
\pgfpathlineto{\pgfqpoint{0cm}{1.552cm}}
\pgfpathclose
\pgfusepath{clip}
\begin{pgfscope}
\begin{pgfscope}
\pgfpathmoveto{\pgfqpoint{0cm}{-0.035cm}}
\pgfpathlineto{\pgfqpoint{1.376cm}{-0.035cm}}
\pgfpathlineto{\pgfqpoint{1.376cm}{1.552cm}}
\pgfpathlineto{\pgfqpoint{0cm}{1.552cm}}
\pgfpathclose
\pgfusepath{clip}
\begin{pgfscope}
\begin{pgfscope}
\pgfsetdash{}{0cm}
\pgfsetlinewidth{0.818mm}
\pgfsetroundcap
\pgfsetroundjoin
\pgfsetmiterlimit{7.0}
\definecolor{eps2pgf_color}{gray}{0}\pgfsetstrokecolor{eps2pgf_color}\pgfsetfillcolor{eps2pgf_color}
\pgfpathmoveto{\pgfqpoint{0.117cm}{1.421cm}}
\pgfpathlineto{\pgfqpoint{0.682cm}{0.671cm}}
\pgfpathlineto{\pgfqpoint{1.246cm}{1.421cm}}
\pgfusepath{stroke}
\end{pgfscope}
\definecolor{eps2pgf_color}{gray}{0}\pgfsetstrokecolor{eps2pgf_color}\pgfsetfillcolor{eps2pgf_color}
\pgfpathmoveto{\pgfqpoint{0.273cm}{1.395cm}}
\pgfpathcurveto{\pgfqpoint{0.273cm}{1.432cm}}{\pgfqpoint{0.259cm}{1.467cm}}{\pgfqpoint{0.233cm}{1.492cm}}
\pgfpathcurveto{\pgfqpoint{0.207cm}{1.518cm}}{\pgfqpoint{0.173cm}{1.532cm}}{\pgfqpoint{0.137cm}{1.532cm}}
\pgfpathcurveto{\pgfqpoint{0.1cm}{1.532cm}}{\pgfqpoint{0.066cm}{1.518cm}}{\pgfqpoint{0.04cm}{1.492cm}}
\pgfpathcurveto{\pgfqpoint{0.014cm}{1.467cm}}{\pgfqpoint{0cm}{1.432cm}}{\pgfqpoint{0cm}{1.395cm}}
\pgfpathcurveto{\pgfqpoint{0cm}{1.359cm}}{\pgfqpoint{0.014cm}{1.324cm}}{\pgfqpoint{0.04cm}{1.299cm}}
\pgfpathcurveto{\pgfqpoint{0.066cm}{1.273cm}}{\pgfqpoint{0.1cm}{1.258cm}}{\pgfqpoint{0.137cm}{1.258cm}}
\pgfpathcurveto{\pgfqpoint{0.173cm}{1.258cm}}{\pgfqpoint{0.207cm}{1.273cm}}{\pgfqpoint{0.233cm}{1.299cm}}
\pgfpathcurveto{\pgfqpoint{0.259cm}{1.324cm}}{\pgfqpoint{0.273cm}{1.359cm}}{\pgfqpoint{0.273cm}{1.395cm}}
\pgfusepath{fill}
\begin{pgfscope}
\pgfsetdash{}{0cm}
\pgfsetlinewidth{0.818mm}
\pgfsetmiterlimit{7.0}
\pgfpathmoveto{\pgfqpoint{0.682cm}{0.671cm}}
\pgfpathlineto{\pgfqpoint{0.679cm}{1.418cm}}
\pgfusepath{stroke}
\end{pgfscope}
\pgfpathmoveto{\pgfqpoint{0.815cm}{1.399cm}}
\pgfpathcurveto{\pgfqpoint{0.815cm}{1.435cm}}{\pgfqpoint{0.801cm}{1.47cm}}{\pgfqpoint{0.775cm}{1.496cm}}
\pgfpathcurveto{\pgfqpoint{0.75cm}{1.521cm}}{\pgfqpoint{0.715cm}{1.536cm}}{\pgfqpoint{0.679cm}{1.536cm}}
\pgfpathcurveto{\pgfqpoint{0.643cm}{1.536cm}}{\pgfqpoint{0.608cm}{1.521cm}}{\pgfqpoint{0.582cm}{1.496cm}}
\pgfpathcurveto{\pgfqpoint{0.557cm}{1.47cm}}{\pgfqpoint{0.542cm}{1.435cm}}{\pgfqpoint{0.542cm}{1.399cm}}
\pgfpathcurveto{\pgfqpoint{0.542cm}{1.363cm}}{\pgfqpoint{0.557cm}{1.328cm}}{\pgfqpoint{0.582cm}{1.302cm}}
\pgfpathcurveto{\pgfqpoint{0.608cm}{1.276cm}}{\pgfqpoint{0.643cm}{1.262cm}}{\pgfqpoint{0.679cm}{1.262cm}}
\pgfpathcurveto{\pgfqpoint{0.715cm}{1.262cm}}{\pgfqpoint{0.75cm}{1.276cm}}{\pgfqpoint{0.775cm}{1.302cm}}
\pgfpathcurveto{\pgfqpoint{0.801cm}{1.328cm}}{\pgfqpoint{0.815cm}{1.363cm}}{\pgfqpoint{0.815cm}{1.399cm}}
\pgfusepath{fill}
\pgfpathmoveto{\pgfqpoint{1.345cm}{1.371cm}}
\pgfpathcurveto{\pgfqpoint{1.345cm}{1.408cm}}{\pgfqpoint{1.331cm}{1.442cm}}{\pgfqpoint{1.305cm}{1.468cm}}
\pgfpathcurveto{\pgfqpoint{1.28cm}{1.494cm}}{\pgfqpoint{1.245cm}{1.508cm}}{\pgfqpoint{1.209cm}{1.508cm}}
\pgfpathcurveto{\pgfqpoint{1.172cm}{1.508cm}}{\pgfqpoint{1.138cm}{1.494cm}}{\pgfqpoint{1.112cm}{1.468cm}}
\pgfpathcurveto{\pgfqpoint{1.087cm}{1.442cm}}{\pgfqpoint{1.072cm}{1.408cm}}{\pgfqpoint{1.072cm}{1.371cm}}
\pgfpathcurveto{\pgfqpoint{1.072cm}{1.335cm}}{\pgfqpoint{1.087cm}{1.3cm}}{\pgfqpoint{1.112cm}{1.274cm}}
\pgfpathcurveto{\pgfqpoint{1.138cm}{1.249cm}}{\pgfqpoint{1.172cm}{1.234cm}}{\pgfqpoint{1.209cm}{1.234cm}}
\pgfpathcurveto{\pgfqpoint{1.245cm}{1.234cm}}{\pgfqpoint{1.28cm}{1.249cm}}{\pgfqpoint{1.305cm}{1.274cm}}
\pgfpathcurveto{\pgfqpoint{1.331cm}{1.3cm}}{\pgfqpoint{1.345cm}{1.335cm}}{\pgfqpoint{1.345cm}{1.371cm}}
\pgfusepath{fill}
\begin{pgfscope}
\pgfsetdash{}{0cm}
\pgfsetlinewidth{0.818mm}
\pgfsetroundcap
\pgfsetmiterlimit{4.0}
\pgfpathmoveto{\pgfqpoint{0.682cm}{0.671cm}}
\pgfpathlineto{\pgfqpoint{0.682cm}{0.042cm}}
\pgfusepath{stroke}
\end{pgfscope}
\end{pgfscope}
\end{pgfscope}
\end{pgfscope}
\end{tikzpicture}}}, \]
\[ X_{M, \varepsilon}^{\!\resizebox{!}{.8em}{
\begin{tikzpicture}
\pgfpathmoveto{\pgfqpoint{0cm}{-0.035cm}}
\pgfpathlineto{\pgfqpoint{1.976cm}{-0.035cm}}
\pgfpathlineto{\pgfqpoint{1.976cm}{1.94cm}}
\pgfpathlineto{\pgfqpoint{0cm}{1.94cm}}
\pgfpathclose
\pgfusepath{clip}
\begin{pgfscope}
\begin{pgfscope}
\pgfpathmoveto{\pgfqpoint{0cm}{-0.035cm}}
\pgfpathlineto{\pgfqpoint{1.976cm}{-0.035cm}}
\pgfpathlineto{\pgfqpoint{1.976cm}{1.94cm}}
\pgfpathlineto{\pgfqpoint{0cm}{1.94cm}}
\pgfpathclose
\pgfusepath{clip}
\begin{pgfscope}
\begin{pgfscope}
\pgfsetdash{}{0cm}
\pgfsetlinewidth{0.818mm}
\pgfsetroundcap
\pgfsetroundjoin
\pgfsetmiterlimit{7.0}
\definecolor{eps2pgf_color}{gray}{0}\pgfsetstrokecolor{eps2pgf_color}\pgfsetfillcolor{eps2pgf_color}
\pgfpathmoveto{\pgfqpoint{0.117cm}{1.815cm}}
\pgfpathlineto{\pgfqpoint{0.682cm}{1.065cm}}
\pgfpathlineto{\pgfqpoint{1.246cm}{1.815cm}}
\pgfusepath{stroke}
\end{pgfscope}
\definecolor{eps2pgf_color}{gray}{0}\pgfsetstrokecolor{eps2pgf_color}\pgfsetfillcolor{eps2pgf_color}
\pgfpathmoveto{\pgfqpoint{0.273cm}{1.789cm}}
\pgfpathcurveto{\pgfqpoint{0.273cm}{1.825cm}}{\pgfqpoint{0.259cm}{1.86cm}}{\pgfqpoint{0.233cm}{1.886cm}}
\pgfpathcurveto{\pgfqpoint{0.207cm}{1.912cm}}{\pgfqpoint{0.173cm}{1.926cm}}{\pgfqpoint{0.137cm}{1.926cm}}
\pgfpathcurveto{\pgfqpoint{0.1cm}{1.926cm}}{\pgfqpoint{0.066cm}{1.912cm}}{\pgfqpoint{0.04cm}{1.886cm}}
\pgfpathcurveto{\pgfqpoint{0.014cm}{1.86cm}}{\pgfqpoint{0cm}{1.825cm}}{\pgfqpoint{0cm}{1.789cm}}
\pgfpathcurveto{\pgfqpoint{0cm}{1.753cm}}{\pgfqpoint{0.014cm}{1.718cm}}{\pgfqpoint{0.04cm}{1.692cm}}
\pgfpathcurveto{\pgfqpoint{0.066cm}{1.667cm}}{\pgfqpoint{0.1cm}{1.652cm}}{\pgfqpoint{0.137cm}{1.652cm}}
\pgfpathcurveto{\pgfqpoint{0.173cm}{1.652cm}}{\pgfqpoint{0.207cm}{1.667cm}}{\pgfqpoint{0.233cm}{1.692cm}}
\pgfpathcurveto{\pgfqpoint{0.259cm}{1.718cm}}{\pgfqpoint{0.273cm}{1.753cm}}{\pgfqpoint{0.273cm}{1.789cm}}
\pgfusepath{fill}
\pgfpathmoveto{\pgfqpoint{1.345cm}{1.765cm}}
\pgfpathcurveto{\pgfqpoint{1.345cm}{1.801cm}}{\pgfqpoint{1.331cm}{1.836cm}}{\pgfqpoint{1.305cm}{1.862cm}}
\pgfpathcurveto{\pgfqpoint{1.28cm}{1.887cm}}{\pgfqpoint{1.245cm}{1.902cm}}{\pgfqpoint{1.209cm}{1.902cm}}
\pgfpathcurveto{\pgfqpoint{1.172cm}{1.902cm}}{\pgfqpoint{1.138cm}{1.887cm}}{\pgfqpoint{1.112cm}{1.862cm}}
\pgfpathcurveto{\pgfqpoint{1.087cm}{1.836cm}}{\pgfqpoint{1.072cm}{1.801cm}}{\pgfqpoint{1.072cm}{1.765cm}}
\pgfpathcurveto{\pgfqpoint{1.072cm}{1.728cm}}{\pgfqpoint{1.087cm}{1.694cm}}{\pgfqpoint{1.112cm}{1.668cm}}
\pgfpathcurveto{\pgfqpoint{1.138cm}{1.642cm}}{\pgfqpoint{1.172cm}{1.628cm}}{\pgfqpoint{1.209cm}{1.628cm}}
\pgfpathcurveto{\pgfqpoint{1.245cm}{1.628cm}}{\pgfqpoint{1.28cm}{1.642cm}}{\pgfqpoint{1.305cm}{1.668cm}}
\pgfpathcurveto{\pgfqpoint{1.331cm}{1.694cm}}{\pgfqpoint{1.345cm}{1.728cm}}{\pgfqpoint{1.345cm}{1.765cm}}
\pgfusepath{fill}
\begin{pgfscope}
\pgfsetdash{}{0cm}
\pgfsetlinewidth{0.818mm}
\pgfsetroundcap
\pgfsetroundjoin
\pgfsetmiterlimit{7.0}
\pgfpathmoveto{\pgfqpoint{0.682cm}{1.065cm}}
\pgfpathlineto{\pgfqpoint{1.246cm}{0.315cm}}
\pgfpathlineto{\pgfqpoint{1.811cm}{1.065cm}}
\pgfusepath{stroke}
\end{pgfscope}
\pgfpathmoveto{\pgfqpoint{1.948cm}{1.065cm}}
\pgfpathcurveto{\pgfqpoint{1.948cm}{1.101cm}}{\pgfqpoint{1.933cm}{1.136cm}}{\pgfqpoint{1.907cm}{1.162cm}}
\pgfpathcurveto{\pgfqpoint{1.882cm}{1.187cm}}{\pgfqpoint{1.847cm}{1.202cm}}{\pgfqpoint{1.811cm}{1.202cm}}
\pgfpathcurveto{\pgfqpoint{1.775cm}{1.202cm}}{\pgfqpoint{1.74cm}{1.187cm}}{\pgfqpoint{1.714cm}{1.162cm}}
\pgfpathcurveto{\pgfqpoint{1.689cm}{1.136cm}}{\pgfqpoint{1.674cm}{1.101cm}}{\pgfqpoint{1.674cm}{1.065cm}}
\pgfpathcurveto{\pgfqpoint{1.674cm}{1.029cm}}{\pgfqpoint{1.689cm}{0.994cm}}{\pgfqpoint{1.714cm}{0.968cm}}
\pgfpathcurveto{\pgfqpoint{1.74cm}{0.942cm}}{\pgfqpoint{1.775cm}{0.928cm}}{\pgfqpoint{1.811cm}{0.928cm}}
\pgfpathcurveto{\pgfqpoint{1.847cm}{0.928cm}}{\pgfqpoint{1.882cm}{0.942cm}}{\pgfqpoint{1.907cm}{0.968cm}}
\pgfpathcurveto{\pgfqpoint{1.933cm}{0.994cm}}{\pgfqpoint{1.948cm}{1.029cm}}{\pgfqpoint{1.948cm}{1.065cm}}
\pgfusepath{fill}
\begin{pgfscope}
\pgfsetdash{}{0cm}
\pgfsetlinewidth{0.818mm}
\pgfsetmiterlimit{7.0}
\pgfpathmoveto{\pgfqpoint{1.246cm}{0.315cm}}
\pgfpathlineto{\pgfqpoint{1.244cm}{1.061cm}}
\pgfusepath{stroke}
\end{pgfscope}
\pgfpathmoveto{\pgfqpoint{1.38cm}{1.065cm}}
\pgfpathcurveto{\pgfqpoint{1.38cm}{1.101cm}}{\pgfqpoint{1.366cm}{1.136cm}}{\pgfqpoint{1.34cm}{1.162cm}}
\pgfpathcurveto{\pgfqpoint{1.315cm}{1.187cm}}{\pgfqpoint{1.28cm}{1.202cm}}{\pgfqpoint{1.244cm}{1.202cm}}
\pgfpathcurveto{\pgfqpoint{1.207cm}{1.202cm}}{\pgfqpoint{1.173cm}{1.187cm}}{\pgfqpoint{1.147cm}{1.162cm}}
\pgfpathcurveto{\pgfqpoint{1.121cm}{1.136cm}}{\pgfqpoint{1.107cm}{1.101cm}}{\pgfqpoint{1.107cm}{1.065cm}}
\pgfpathcurveto{\pgfqpoint{1.107cm}{1.029cm}}{\pgfqpoint{1.121cm}{0.994cm}}{\pgfqpoint{1.147cm}{0.968cm}}
\pgfpathcurveto{\pgfqpoint{1.173cm}{0.942cm}}{\pgfqpoint{1.207cm}{0.928cm}}{\pgfqpoint{1.244cm}{0.928cm}}
\pgfpathcurveto{\pgfqpoint{1.28cm}{0.928cm}}{\pgfqpoint{1.315cm}{0.942cm}}{\pgfqpoint{1.34cm}{0.968cm}}
\pgfpathcurveto{\pgfqpoint{1.366cm}{0.994cm}}{\pgfqpoint{1.38cm}{1.029cm}}{\pgfqpoint{1.38cm}{1.065cm}}
\pgfusepath{fill}
\begin{pgfscope}
\pgfsetdash{}{0cm}
\pgfsetlinewidth{0.818mm}
\pgfsetmiterlimit{4.0}
\pgfpathmoveto{\pgfqpoint{1.383cm}{0.178cm}}
\pgfpathcurveto{\pgfqpoint{1.383cm}{0.214cm}}{\pgfqpoint{1.369cm}{0.249cm}}{\pgfqpoint{1.343cm}{0.275cm}}
\pgfpathcurveto{\pgfqpoint{1.317cm}{0.3cm}}{\pgfqpoint{1.283cm}{0.315cm}}{\pgfqpoint{1.246cm}{0.315cm}}
\pgfpathcurveto{\pgfqpoint{1.21cm}{0.315cm}}{\pgfqpoint{1.175cm}{0.3cm}}{\pgfqpoint{1.15cm}{0.275cm}}
\pgfpathcurveto{\pgfqpoint{1.124cm}{0.249cm}}{\pgfqpoint{1.11cm}{0.214cm}}{\pgfqpoint{1.11cm}{0.178cm}}
\pgfpathcurveto{\pgfqpoint{1.11cm}{0.141cm}}{\pgfqpoint{1.124cm}{0.107cm}}{\pgfqpoint{1.15cm}{0.081cm}}
\pgfpathcurveto{\pgfqpoint{1.175cm}{0.055cm}}{\pgfqpoint{1.21cm}{0.041cm}}{\pgfqpoint{1.246cm}{0.041cm}}
\pgfpathcurveto{\pgfqpoint{1.283cm}{0.041cm}}{\pgfqpoint{1.317cm}{0.055cm}}{\pgfqpoint{1.343cm}{0.081cm}}
\pgfpathcurveto{\pgfqpoint{1.369cm}{0.107cm}}{\pgfqpoint{1.383cm}{0.141cm}}{\pgfqpoint{1.383cm}{0.178cm}}
\pgfusepath{stroke}
\end{pgfscope}
\end{pgfscope}
\end{pgfscope}
\end{pgfscope}
\end{tikzpicture}}} \assign 9 \llbracket X_{M, \varepsilon}^2
   \rrbracket \circ \Q_{\varepsilon}^{- 1} \llbracket X_{M, \varepsilon}^2
   \rrbracket - 3 b_{M, \varepsilon}, \]
\[ \tilde{X}_{M, \varepsilon}^{\!\resizebox{!}{.8em}{
\begin{tikzpicture}
\pgfpathmoveto{\pgfqpoint{0cm}{-0.035cm}}
\pgfpathlineto{\pgfqpoint{1.976cm}{-0.035cm}}
\pgfpathlineto{\pgfqpoint{1.976cm}{1.94cm}}
\pgfpathlineto{\pgfqpoint{0cm}{1.94cm}}
\pgfpathclose
\pgfusepath{clip}
\begin{pgfscope}
\begin{pgfscope}
\pgfpathmoveto{\pgfqpoint{0cm}{-0.035cm}}
\pgfpathlineto{\pgfqpoint{1.976cm}{-0.035cm}}
\pgfpathlineto{\pgfqpoint{1.976cm}{1.94cm}}
\pgfpathlineto{\pgfqpoint{0cm}{1.94cm}}
\pgfpathclose
\pgfusepath{clip}
\begin{pgfscope}
\begin{pgfscope}
\pgfsetdash{}{0cm}
\pgfsetlinewidth{0.818mm}
\pgfsetroundcap
\pgfsetroundjoin
\pgfsetmiterlimit{7.0}
\definecolor{eps2pgf_color}{gray}{0}\pgfsetstrokecolor{eps2pgf_color}\pgfsetfillcolor{eps2pgf_color}
\pgfpathmoveto{\pgfqpoint{0.117cm}{1.815cm}}
\pgfpathlineto{\pgfqpoint{0.682cm}{1.065cm}}
\pgfpathlineto{\pgfqpoint{1.246cm}{1.815cm}}
\pgfusepath{stroke}
\end{pgfscope}
\definecolor{eps2pgf_color}{gray}{0}\pgfsetstrokecolor{eps2pgf_color}\pgfsetfillcolor{eps2pgf_color}
\pgfpathmoveto{\pgfqpoint{0.273cm}{1.789cm}}
\pgfpathcurveto{\pgfqpoint{0.273cm}{1.825cm}}{\pgfqpoint{0.259cm}{1.86cm}}{\pgfqpoint{0.233cm}{1.886cm}}
\pgfpathcurveto{\pgfqpoint{0.207cm}{1.912cm}}{\pgfqpoint{0.173cm}{1.926cm}}{\pgfqpoint{0.137cm}{1.926cm}}
\pgfpathcurveto{\pgfqpoint{0.1cm}{1.926cm}}{\pgfqpoint{0.066cm}{1.912cm}}{\pgfqpoint{0.04cm}{1.886cm}}
\pgfpathcurveto{\pgfqpoint{0.014cm}{1.86cm}}{\pgfqpoint{0cm}{1.825cm}}{\pgfqpoint{0cm}{1.789cm}}
\pgfpathcurveto{\pgfqpoint{0cm}{1.753cm}}{\pgfqpoint{0.014cm}{1.718cm}}{\pgfqpoint{0.04cm}{1.692cm}}
\pgfpathcurveto{\pgfqpoint{0.066cm}{1.667cm}}{\pgfqpoint{0.1cm}{1.652cm}}{\pgfqpoint{0.137cm}{1.652cm}}
\pgfpathcurveto{\pgfqpoint{0.173cm}{1.652cm}}{\pgfqpoint{0.207cm}{1.667cm}}{\pgfqpoint{0.233cm}{1.692cm}}
\pgfpathcurveto{\pgfqpoint{0.259cm}{1.718cm}}{\pgfqpoint{0.273cm}{1.753cm}}{\pgfqpoint{0.273cm}{1.789cm}}
\pgfusepath{fill}
\pgfpathmoveto{\pgfqpoint{1.345cm}{1.765cm}}
\pgfpathcurveto{\pgfqpoint{1.345cm}{1.801cm}}{\pgfqpoint{1.331cm}{1.836cm}}{\pgfqpoint{1.305cm}{1.862cm}}
\pgfpathcurveto{\pgfqpoint{1.28cm}{1.887cm}}{\pgfqpoint{1.245cm}{1.902cm}}{\pgfqpoint{1.209cm}{1.902cm}}
\pgfpathcurveto{\pgfqpoint{1.172cm}{1.902cm}}{\pgfqpoint{1.138cm}{1.887cm}}{\pgfqpoint{1.112cm}{1.862cm}}
\pgfpathcurveto{\pgfqpoint{1.087cm}{1.836cm}}{\pgfqpoint{1.072cm}{1.801cm}}{\pgfqpoint{1.072cm}{1.765cm}}
\pgfpathcurveto{\pgfqpoint{1.072cm}{1.728cm}}{\pgfqpoint{1.087cm}{1.694cm}}{\pgfqpoint{1.112cm}{1.668cm}}
\pgfpathcurveto{\pgfqpoint{1.138cm}{1.642cm}}{\pgfqpoint{1.172cm}{1.628cm}}{\pgfqpoint{1.209cm}{1.628cm}}
\pgfpathcurveto{\pgfqpoint{1.245cm}{1.628cm}}{\pgfqpoint{1.28cm}{1.642cm}}{\pgfqpoint{1.305cm}{1.668cm}}
\pgfpathcurveto{\pgfqpoint{1.331cm}{1.694cm}}{\pgfqpoint{1.345cm}{1.728cm}}{\pgfqpoint{1.345cm}{1.765cm}}
\pgfusepath{fill}
\begin{pgfscope}
\pgfsetdash{}{0cm}
\pgfsetlinewidth{0.818mm}
\pgfsetroundcap
\pgfsetroundjoin
\pgfsetmiterlimit{7.0}
\pgfpathmoveto{\pgfqpoint{0.682cm}{1.065cm}}
\pgfpathlineto{\pgfqpoint{1.246cm}{0.315cm}}
\pgfpathlineto{\pgfqpoint{1.811cm}{1.065cm}}
\pgfusepath{stroke}
\end{pgfscope}
\pgfpathmoveto{\pgfqpoint{1.948cm}{1.065cm}}
\pgfpathcurveto{\pgfqpoint{1.948cm}{1.101cm}}{\pgfqpoint{1.933cm}{1.136cm}}{\pgfqpoint{1.907cm}{1.162cm}}
\pgfpathcurveto{\pgfqpoint{1.882cm}{1.187cm}}{\pgfqpoint{1.847cm}{1.202cm}}{\pgfqpoint{1.811cm}{1.202cm}}
\pgfpathcurveto{\pgfqpoint{1.775cm}{1.202cm}}{\pgfqpoint{1.74cm}{1.187cm}}{\pgfqpoint{1.714cm}{1.162cm}}
\pgfpathcurveto{\pgfqpoint{1.689cm}{1.136cm}}{\pgfqpoint{1.674cm}{1.101cm}}{\pgfqpoint{1.674cm}{1.065cm}}
\pgfpathcurveto{\pgfqpoint{1.674cm}{1.029cm}}{\pgfqpoint{1.689cm}{0.994cm}}{\pgfqpoint{1.714cm}{0.968cm}}
\pgfpathcurveto{\pgfqpoint{1.74cm}{0.942cm}}{\pgfqpoint{1.775cm}{0.928cm}}{\pgfqpoint{1.811cm}{0.928cm}}
\pgfpathcurveto{\pgfqpoint{1.847cm}{0.928cm}}{\pgfqpoint{1.882cm}{0.942cm}}{\pgfqpoint{1.907cm}{0.968cm}}
\pgfpathcurveto{\pgfqpoint{1.933cm}{0.994cm}}{\pgfqpoint{1.948cm}{1.029cm}}{\pgfqpoint{1.948cm}{1.065cm}}
\pgfusepath{fill}
\begin{pgfscope}
\pgfsetdash{}{0cm}
\pgfsetlinewidth{0.818mm}
\pgfsetmiterlimit{7.0}
\pgfpathmoveto{\pgfqpoint{1.246cm}{0.315cm}}
\pgfpathlineto{\pgfqpoint{1.244cm}{1.061cm}}
\pgfusepath{stroke}
\end{pgfscope}
\pgfpathmoveto{\pgfqpoint{1.38cm}{1.065cm}}
\pgfpathcurveto{\pgfqpoint{1.38cm}{1.101cm}}{\pgfqpoint{1.366cm}{1.136cm}}{\pgfqpoint{1.34cm}{1.162cm}}
\pgfpathcurveto{\pgfqpoint{1.315cm}{1.187cm}}{\pgfqpoint{1.28cm}{1.202cm}}{\pgfqpoint{1.244cm}{1.202cm}}
\pgfpathcurveto{\pgfqpoint{1.207cm}{1.202cm}}{\pgfqpoint{1.173cm}{1.187cm}}{\pgfqpoint{1.147cm}{1.162cm}}
\pgfpathcurveto{\pgfqpoint{1.121cm}{1.136cm}}{\pgfqpoint{1.107cm}{1.101cm}}{\pgfqpoint{1.107cm}{1.065cm}}
\pgfpathcurveto{\pgfqpoint{1.107cm}{1.029cm}}{\pgfqpoint{1.121cm}{0.994cm}}{\pgfqpoint{1.147cm}{0.968cm}}
\pgfpathcurveto{\pgfqpoint{1.173cm}{0.942cm}}{\pgfqpoint{1.207cm}{0.928cm}}{\pgfqpoint{1.244cm}{0.928cm}}
\pgfpathcurveto{\pgfqpoint{1.28cm}{0.928cm}}{\pgfqpoint{1.315cm}{0.942cm}}{\pgfqpoint{1.34cm}{0.968cm}}
\pgfpathcurveto{\pgfqpoint{1.366cm}{0.994cm}}{\pgfqpoint{1.38cm}{1.029cm}}{\pgfqpoint{1.38cm}{1.065cm}}
\pgfusepath{fill}
\begin{pgfscope}
\pgfsetdash{}{0cm}
\pgfsetlinewidth{0.818mm}
\pgfsetmiterlimit{4.0}
\pgfpathmoveto{\pgfqpoint{1.383cm}{0.178cm}}
\pgfpathcurveto{\pgfqpoint{1.383cm}{0.214cm}}{\pgfqpoint{1.369cm}{0.249cm}}{\pgfqpoint{1.343cm}{0.275cm}}
\pgfpathcurveto{\pgfqpoint{1.317cm}{0.3cm}}{\pgfqpoint{1.283cm}{0.315cm}}{\pgfqpoint{1.246cm}{0.315cm}}
\pgfpathcurveto{\pgfqpoint{1.21cm}{0.315cm}}{\pgfqpoint{1.175cm}{0.3cm}}{\pgfqpoint{1.15cm}{0.275cm}}
\pgfpathcurveto{\pgfqpoint{1.124cm}{0.249cm}}{\pgfqpoint{1.11cm}{0.214cm}}{\pgfqpoint{1.11cm}{0.178cm}}
\pgfpathcurveto{\pgfqpoint{1.11cm}{0.141cm}}{\pgfqpoint{1.124cm}{0.107cm}}{\pgfqpoint{1.15cm}{0.081cm}}
\pgfpathcurveto{\pgfqpoint{1.175cm}{0.055cm}}{\pgfqpoint{1.21cm}{0.041cm}}{\pgfqpoint{1.246cm}{0.041cm}}
\pgfpathcurveto{\pgfqpoint{1.283cm}{0.041cm}}{\pgfqpoint{1.317cm}{0.055cm}}{\pgfqpoint{1.343cm}{0.081cm}}
\pgfpathcurveto{\pgfqpoint{1.369cm}{0.107cm}}{\pgfqpoint{1.383cm}{0.141cm}}{\pgfqpoint{1.383cm}{0.178cm}}
\pgfusepath{stroke}
\end{pgfscope}
\end{pgfscope}
\end{pgfscope}
\end{pgfscope}
\end{tikzpicture}}} = 9 \llbracket X_{M, \varepsilon}^2
   \rrbracket \circ X_{M, \varepsilon}^{\!\resizebox{0.6em}{!}{
\begin{tikzpicture}
\pgfpathmoveto{\pgfqpoint{0cm}{0cm}}
\pgfpathlineto{\pgfqpoint{1.376cm}{0cm}}
\pgfpathlineto{\pgfqpoint{1.376cm}{1.588cm}}
\pgfpathlineto{\pgfqpoint{0cm}{1.588cm}}
\pgfpathclose
\pgfusepath{clip}
\begin{pgfscope}
\begin{pgfscope}
\pgfpathmoveto{\pgfqpoint{0cm}{0cm}}
\pgfpathlineto{\pgfqpoint{1.376cm}{0cm}}
\pgfpathlineto{\pgfqpoint{1.376cm}{1.588cm}}
\pgfpathlineto{\pgfqpoint{0cm}{1.588cm}}
\pgfpathclose
\pgfusepath{clip}
\begin{pgfscope}
\begin{pgfscope}
\definecolor{eps2pgf_color}{gray}{0.976471}\pgfsetstrokecolor{eps2pgf_color}\pgfsetfillcolor{eps2pgf_color}
\pgfpathmoveto{\pgfqpoint{0cm}{0cm}}
\pgfpathlineto{\pgfqpoint{1.376cm}{0cm}}
\pgfpathlineto{\pgfqpoint{1.376cm}{1.588cm}}
\pgfpathlineto{\pgfqpoint{0cm}{1.588cm}}
\pgfpathclose
\pgfusepath{fill}
\end{pgfscope}
\begin{pgfscope}
\pgfsetdash{}{0cm}
\pgfsetlinewidth{0.818mm}
\pgfsetroundcap
\pgfsetroundjoin
\pgfsetmiterlimit{7.0}
\definecolor{eps2pgf_color}{gray}{0}\pgfsetstrokecolor{eps2pgf_color}\pgfsetfillcolor{eps2pgf_color}
\pgfpathmoveto{\pgfqpoint{0.117cm}{1.476cm}}
\pgfpathlineto{\pgfqpoint{0.682cm}{0.726cm}}
\pgfpathlineto{\pgfqpoint{1.246cm}{1.476cm}}
\pgfusepath{stroke}
\end{pgfscope}
\definecolor{eps2pgf_color}{gray}{0}\pgfsetstrokecolor{eps2pgf_color}\pgfsetfillcolor{eps2pgf_color}
\pgfpathmoveto{\pgfqpoint{0.273cm}{1.451cm}}
\pgfpathcurveto{\pgfqpoint{0.273cm}{1.487cm}}{\pgfqpoint{0.259cm}{1.522cm}}{\pgfqpoint{0.233cm}{1.547cm}}
\pgfpathcurveto{\pgfqpoint{0.207cm}{1.573cm}}{\pgfqpoint{0.173cm}{1.588cm}}{\pgfqpoint{0.137cm}{1.588cm}}
\pgfpathcurveto{\pgfqpoint{0.1cm}{1.588cm}}{\pgfqpoint{0.066cm}{1.573cm}}{\pgfqpoint{0.04cm}{1.547cm}}
\pgfpathcurveto{\pgfqpoint{0.014cm}{1.522cm}}{\pgfqpoint{0cm}{1.487cm}}{\pgfqpoint{0cm}{1.451cm}}
\pgfpathcurveto{\pgfqpoint{0cm}{1.414cm}}{\pgfqpoint{0.014cm}{1.379cm}}{\pgfqpoint{0.04cm}{1.354cm}}
\pgfpathcurveto{\pgfqpoint{0.066cm}{1.328cm}}{\pgfqpoint{0.1cm}{1.314cm}}{\pgfqpoint{0.137cm}{1.314cm}}
\pgfpathcurveto{\pgfqpoint{0.173cm}{1.314cm}}{\pgfqpoint{0.207cm}{1.328cm}}{\pgfqpoint{0.233cm}{1.354cm}}
\pgfpathcurveto{\pgfqpoint{0.259cm}{1.379cm}}{\pgfqpoint{0.273cm}{1.414cm}}{\pgfqpoint{0.273cm}{1.451cm}}
\pgfusepath{fill}
\pgfpathmoveto{\pgfqpoint{1.345cm}{1.426cm}}
\pgfpathcurveto{\pgfqpoint{1.345cm}{1.463cm}}{\pgfqpoint{1.331cm}{1.497cm}}{\pgfqpoint{1.305cm}{1.523cm}}
\pgfpathcurveto{\pgfqpoint{1.28cm}{1.549cm}}{\pgfqpoint{1.245cm}{1.563cm}}{\pgfqpoint{1.209cm}{1.563cm}}
\pgfpathcurveto{\pgfqpoint{1.172cm}{1.563cm}}{\pgfqpoint{1.138cm}{1.549cm}}{\pgfqpoint{1.112cm}{1.523cm}}
\pgfpathcurveto{\pgfqpoint{1.087cm}{1.497cm}}{\pgfqpoint{1.072cm}{1.463cm}}{\pgfqpoint{1.072cm}{1.426cm}}
\pgfpathcurveto{\pgfqpoint{1.072cm}{1.39cm}}{\pgfqpoint{1.087cm}{1.355cm}}{\pgfqpoint{1.112cm}{1.329cm}}
\pgfpathcurveto{\pgfqpoint{1.138cm}{1.304cm}}{\pgfqpoint{1.172cm}{1.289cm}}{\pgfqpoint{1.209cm}{1.289cm}}
\pgfpathcurveto{\pgfqpoint{1.245cm}{1.289cm}}{\pgfqpoint{1.28cm}{1.304cm}}{\pgfqpoint{1.305cm}{1.329cm}}
\pgfpathcurveto{\pgfqpoint{1.331cm}{1.355cm}}{\pgfqpoint{1.345cm}{1.39cm}}{\pgfqpoint{1.345cm}{1.426cm}}
\pgfusepath{fill}
\begin{pgfscope}
\pgfsetdash{}{0cm}
\pgfsetlinewidth{0.818mm}
\pgfsetroundcap
\pgfsetmiterlimit{4.0}
\pgfpathmoveto{\pgfqpoint{0.682cm}{0.726cm}}
\pgfpathlineto{\pgfqpoint{0.682cm}{0.097cm}}
\pgfusepath{stroke}
\end{pgfscope}
\end{pgfscope}
\end{pgfscope}
\end{pgfscope}
\end{tikzpicture}}} - 3 \tilde{b}_{M,
   \varepsilon} (t), \qquad X_{M, \varepsilon}^{\!\resizebox{!}{.8em}{
\begin{tikzpicture}
\pgfpathmoveto{\pgfqpoint{0cm}{-0.035cm}}
\pgfpathlineto{\pgfqpoint{1.976cm}{-0.035cm}}
\pgfpathlineto{\pgfqpoint{1.976cm}{1.94cm}}
\pgfpathlineto{\pgfqpoint{0cm}{1.94cm}}
\pgfpathclose
\pgfusepath{clip}
\begin{pgfscope}
\begin{pgfscope}
\pgfpathmoveto{\pgfqpoint{0cm}{-0.035cm}}
\pgfpathlineto{\pgfqpoint{1.976cm}{-0.035cm}}
\pgfpathlineto{\pgfqpoint{1.976cm}{1.94cm}}
\pgfpathlineto{\pgfqpoint{0cm}{1.94cm}}
\pgfpathclose
\pgfusepath{clip}
\begin{pgfscope}
\begin{pgfscope}
\pgfsetdash{}{0cm}
\pgfsetlinewidth{0.818mm}
\pgfsetroundcap
\pgfsetroundjoin
\pgfsetmiterlimit{7.0}
\definecolor{eps2pgf_color}{gray}{0}\pgfsetstrokecolor{eps2pgf_color}\pgfsetfillcolor{eps2pgf_color}
\pgfpathmoveto{\pgfqpoint{0.117cm}{1.815cm}}
\pgfpathlineto{\pgfqpoint{0.682cm}{1.065cm}}
\pgfpathlineto{\pgfqpoint{1.246cm}{1.815cm}}
\pgfusepath{stroke}
\end{pgfscope}
\definecolor{eps2pgf_color}{gray}{0}\pgfsetstrokecolor{eps2pgf_color}\pgfsetfillcolor{eps2pgf_color}
\pgfpathmoveto{\pgfqpoint{0.273cm}{1.789cm}}
\pgfpathcurveto{\pgfqpoint{0.273cm}{1.825cm}}{\pgfqpoint{0.259cm}{1.86cm}}{\pgfqpoint{0.233cm}{1.886cm}}
\pgfpathcurveto{\pgfqpoint{0.207cm}{1.912cm}}{\pgfqpoint{0.173cm}{1.926cm}}{\pgfqpoint{0.137cm}{1.926cm}}
\pgfpathcurveto{\pgfqpoint{0.1cm}{1.926cm}}{\pgfqpoint{0.066cm}{1.912cm}}{\pgfqpoint{0.04cm}{1.886cm}}
\pgfpathcurveto{\pgfqpoint{0.014cm}{1.86cm}}{\pgfqpoint{0cm}{1.825cm}}{\pgfqpoint{0cm}{1.789cm}}
\pgfpathcurveto{\pgfqpoint{0cm}{1.753cm}}{\pgfqpoint{0.014cm}{1.718cm}}{\pgfqpoint{0.04cm}{1.692cm}}
\pgfpathcurveto{\pgfqpoint{0.066cm}{1.667cm}}{\pgfqpoint{0.1cm}{1.652cm}}{\pgfqpoint{0.137cm}{1.652cm}}
\pgfpathcurveto{\pgfqpoint{0.173cm}{1.652cm}}{\pgfqpoint{0.207cm}{1.667cm}}{\pgfqpoint{0.233cm}{1.692cm}}
\pgfpathcurveto{\pgfqpoint{0.259cm}{1.718cm}}{\pgfqpoint{0.273cm}{1.753cm}}{\pgfqpoint{0.273cm}{1.789cm}}
\pgfusepath{fill}
\begin{pgfscope}
\pgfsetdash{}{0cm}
\pgfsetlinewidth{0.818mm}
\pgfsetmiterlimit{7.0}
\pgfpathmoveto{\pgfqpoint{0.682cm}{1.065cm}}
\pgfpathlineto{\pgfqpoint{0.679cm}{1.812cm}}
\pgfusepath{stroke}
\end{pgfscope}
\pgfpathmoveto{\pgfqpoint{0.815cm}{1.793cm}}
\pgfpathcurveto{\pgfqpoint{0.815cm}{1.829cm}}{\pgfqpoint{0.801cm}{1.864cm}}{\pgfqpoint{0.775cm}{1.89cm}}
\pgfpathcurveto{\pgfqpoint{0.75cm}{1.915cm}}{\pgfqpoint{0.715cm}{1.93cm}}{\pgfqpoint{0.679cm}{1.93cm}}
\pgfpathcurveto{\pgfqpoint{0.643cm}{1.93cm}}{\pgfqpoint{0.608cm}{1.915cm}}{\pgfqpoint{0.582cm}{1.89cm}}
\pgfpathcurveto{\pgfqpoint{0.557cm}{1.864cm}}{\pgfqpoint{0.542cm}{1.829cm}}{\pgfqpoint{0.542cm}{1.793cm}}
\pgfpathcurveto{\pgfqpoint{0.542cm}{1.756cm}}{\pgfqpoint{0.557cm}{1.722cm}}{\pgfqpoint{0.582cm}{1.696cm}}
\pgfpathcurveto{\pgfqpoint{0.608cm}{1.67cm}}{\pgfqpoint{0.643cm}{1.656cm}}{\pgfqpoint{0.679cm}{1.656cm}}
\pgfpathcurveto{\pgfqpoint{0.715cm}{1.656cm}}{\pgfqpoint{0.75cm}{1.67cm}}{\pgfqpoint{0.775cm}{1.696cm}}
\pgfpathcurveto{\pgfqpoint{0.801cm}{1.722cm}}{\pgfqpoint{0.815cm}{1.756cm}}{\pgfqpoint{0.815cm}{1.793cm}}
\pgfusepath{fill}
\pgfpathmoveto{\pgfqpoint{1.345cm}{1.765cm}}
\pgfpathcurveto{\pgfqpoint{1.345cm}{1.801cm}}{\pgfqpoint{1.331cm}{1.836cm}}{\pgfqpoint{1.305cm}{1.862cm}}
\pgfpathcurveto{\pgfqpoint{1.28cm}{1.887cm}}{\pgfqpoint{1.245cm}{1.902cm}}{\pgfqpoint{1.209cm}{1.902cm}}
\pgfpathcurveto{\pgfqpoint{1.172cm}{1.902cm}}{\pgfqpoint{1.138cm}{1.887cm}}{\pgfqpoint{1.112cm}{1.862cm}}
\pgfpathcurveto{\pgfqpoint{1.087cm}{1.836cm}}{\pgfqpoint{1.072cm}{1.801cm}}{\pgfqpoint{1.072cm}{1.765cm}}
\pgfpathcurveto{\pgfqpoint{1.072cm}{1.728cm}}{\pgfqpoint{1.087cm}{1.694cm}}{\pgfqpoint{1.112cm}{1.668cm}}
\pgfpathcurveto{\pgfqpoint{1.138cm}{1.642cm}}{\pgfqpoint{1.172cm}{1.628cm}}{\pgfqpoint{1.209cm}{1.628cm}}
\pgfpathcurveto{\pgfqpoint{1.245cm}{1.628cm}}{\pgfqpoint{1.28cm}{1.642cm}}{\pgfqpoint{1.305cm}{1.668cm}}
\pgfpathcurveto{\pgfqpoint{1.331cm}{1.694cm}}{\pgfqpoint{1.345cm}{1.728cm}}{\pgfqpoint{1.345cm}{1.765cm}}
\pgfusepath{fill}
\begin{pgfscope}
\pgfsetdash{}{0cm}
\pgfsetlinewidth{0.818mm}
\pgfsetroundcap
\pgfsetroundjoin
\pgfsetmiterlimit{7.0}
\pgfpathmoveto{\pgfqpoint{0.682cm}{1.065cm}}
\pgfpathlineto{\pgfqpoint{1.246cm}{0.315cm}}
\pgfpathlineto{\pgfqpoint{1.811cm}{1.065cm}}
\pgfusepath{stroke}
\end{pgfscope}
\pgfpathmoveto{\pgfqpoint{1.948cm}{1.065cm}}
\pgfpathcurveto{\pgfqpoint{1.948cm}{1.101cm}}{\pgfqpoint{1.933cm}{1.136cm}}{\pgfqpoint{1.907cm}{1.162cm}}
\pgfpathcurveto{\pgfqpoint{1.882cm}{1.187cm}}{\pgfqpoint{1.847cm}{1.202cm}}{\pgfqpoint{1.811cm}{1.202cm}}
\pgfpathcurveto{\pgfqpoint{1.775cm}{1.202cm}}{\pgfqpoint{1.74cm}{1.187cm}}{\pgfqpoint{1.714cm}{1.162cm}}
\pgfpathcurveto{\pgfqpoint{1.689cm}{1.136cm}}{\pgfqpoint{1.674cm}{1.101cm}}{\pgfqpoint{1.674cm}{1.065cm}}
\pgfpathcurveto{\pgfqpoint{1.674cm}{1.029cm}}{\pgfqpoint{1.689cm}{0.994cm}}{\pgfqpoint{1.714cm}{0.968cm}}
\pgfpathcurveto{\pgfqpoint{1.74cm}{0.942cm}}{\pgfqpoint{1.775cm}{0.928cm}}{\pgfqpoint{1.811cm}{0.928cm}}
\pgfpathcurveto{\pgfqpoint{1.847cm}{0.928cm}}{\pgfqpoint{1.882cm}{0.942cm}}{\pgfqpoint{1.907cm}{0.968cm}}
\pgfpathcurveto{\pgfqpoint{1.933cm}{0.994cm}}{\pgfqpoint{1.948cm}{1.029cm}}{\pgfqpoint{1.948cm}{1.065cm}}
\pgfusepath{fill}
\begin{pgfscope}
\pgfsetdash{}{0cm}
\pgfsetlinewidth{0.818mm}
\pgfsetmiterlimit{7.0}
\pgfpathmoveto{\pgfqpoint{1.246cm}{0.315cm}}
\pgfpathlineto{\pgfqpoint{1.244cm}{1.061cm}}
\pgfusepath{stroke}
\end{pgfscope}
\pgfpathmoveto{\pgfqpoint{1.38cm}{1.065cm}}
\pgfpathcurveto{\pgfqpoint{1.38cm}{1.101cm}}{\pgfqpoint{1.366cm}{1.136cm}}{\pgfqpoint{1.34cm}{1.162cm}}
\pgfpathcurveto{\pgfqpoint{1.315cm}{1.187cm}}{\pgfqpoint{1.28cm}{1.202cm}}{\pgfqpoint{1.244cm}{1.202cm}}
\pgfpathcurveto{\pgfqpoint{1.207cm}{1.202cm}}{\pgfqpoint{1.173cm}{1.187cm}}{\pgfqpoint{1.147cm}{1.162cm}}
\pgfpathcurveto{\pgfqpoint{1.121cm}{1.136cm}}{\pgfqpoint{1.107cm}{1.101cm}}{\pgfqpoint{1.107cm}{1.065cm}}
\pgfpathcurveto{\pgfqpoint{1.107cm}{1.029cm}}{\pgfqpoint{1.121cm}{0.994cm}}{\pgfqpoint{1.147cm}{0.968cm}}
\pgfpathcurveto{\pgfqpoint{1.173cm}{0.942cm}}{\pgfqpoint{1.207cm}{0.928cm}}{\pgfqpoint{1.244cm}{0.928cm}}
\pgfpathcurveto{\pgfqpoint{1.28cm}{0.928cm}}{\pgfqpoint{1.315cm}{0.942cm}}{\pgfqpoint{1.34cm}{0.968cm}}
\pgfpathcurveto{\pgfqpoint{1.366cm}{0.994cm}}{\pgfqpoint{1.38cm}{1.029cm}}{\pgfqpoint{1.38cm}{1.065cm}}
\pgfusepath{fill}
\begin{pgfscope}
\pgfsetdash{}{0cm}
\pgfsetlinewidth{0.818mm}
\pgfsetmiterlimit{4.0}
\pgfpathmoveto{\pgfqpoint{1.383cm}{0.178cm}}
\pgfpathcurveto{\pgfqpoint{1.383cm}{0.214cm}}{\pgfqpoint{1.369cm}{0.249cm}}{\pgfqpoint{1.343cm}{0.275cm}}
\pgfpathcurveto{\pgfqpoint{1.317cm}{0.3cm}}{\pgfqpoint{1.283cm}{0.315cm}}{\pgfqpoint{1.246cm}{0.315cm}}
\pgfpathcurveto{\pgfqpoint{1.21cm}{0.315cm}}{\pgfqpoint{1.175cm}{0.3cm}}{\pgfqpoint{1.15cm}{0.275cm}}
\pgfpathcurveto{\pgfqpoint{1.124cm}{0.249cm}}{\pgfqpoint{1.11cm}{0.214cm}}{\pgfqpoint{1.11cm}{0.178cm}}
\pgfpathcurveto{\pgfqpoint{1.11cm}{0.141cm}}{\pgfqpoint{1.124cm}{0.107cm}}{\pgfqpoint{1.15cm}{0.081cm}}
\pgfpathcurveto{\pgfqpoint{1.175cm}{0.055cm}}{\pgfqpoint{1.21cm}{0.041cm}}{\pgfqpoint{1.246cm}{0.041cm}}
\pgfpathcurveto{\pgfqpoint{1.283cm}{0.041cm}}{\pgfqpoint{1.317cm}{0.055cm}}{\pgfqpoint{1.343cm}{0.081cm}}
\pgfpathcurveto{\pgfqpoint{1.369cm}{0.107cm}}{\pgfqpoint{1.383cm}{0.141cm}}{\pgfqpoint{1.383cm}{0.178cm}}
\pgfusepath{stroke}
\end{pgfscope}
\end{pgfscope}
\end{pgfscope}
\end{pgfscope}
\end{tikzpicture}}} = 3 \llbracket
   X_{M, \varepsilon}^2 \rrbracket \circ X_{M, \varepsilon}^{\!\resizebox{0.6em}{!}{
\begin{tikzpicture}
\pgfpathmoveto{\pgfqpoint{0cm}{-0.035cm}}
\pgfpathlineto{\pgfqpoint{1.376cm}{-0.035cm}}
\pgfpathlineto{\pgfqpoint{1.376cm}{1.552cm}}
\pgfpathlineto{\pgfqpoint{0cm}{1.552cm}}
\pgfpathclose
\pgfusepath{clip}
\begin{pgfscope}
\begin{pgfscope}
\pgfpathmoveto{\pgfqpoint{0cm}{-0.035cm}}
\pgfpathlineto{\pgfqpoint{1.376cm}{-0.035cm}}
\pgfpathlineto{\pgfqpoint{1.376cm}{1.552cm}}
\pgfpathlineto{\pgfqpoint{0cm}{1.552cm}}
\pgfpathclose
\pgfusepath{clip}
\begin{pgfscope}
\begin{pgfscope}
\pgfsetdash{}{0cm}
\pgfsetlinewidth{0.818mm}
\pgfsetroundcap
\pgfsetroundjoin
\pgfsetmiterlimit{7.0}
\definecolor{eps2pgf_color}{gray}{0}\pgfsetstrokecolor{eps2pgf_color}\pgfsetfillcolor{eps2pgf_color}
\pgfpathmoveto{\pgfqpoint{0.117cm}{1.421cm}}
\pgfpathlineto{\pgfqpoint{0.682cm}{0.671cm}}
\pgfpathlineto{\pgfqpoint{1.246cm}{1.421cm}}
\pgfusepath{stroke}
\end{pgfscope}
\definecolor{eps2pgf_color}{gray}{0}\pgfsetstrokecolor{eps2pgf_color}\pgfsetfillcolor{eps2pgf_color}
\pgfpathmoveto{\pgfqpoint{0.273cm}{1.395cm}}
\pgfpathcurveto{\pgfqpoint{0.273cm}{1.432cm}}{\pgfqpoint{0.259cm}{1.467cm}}{\pgfqpoint{0.233cm}{1.492cm}}
\pgfpathcurveto{\pgfqpoint{0.207cm}{1.518cm}}{\pgfqpoint{0.173cm}{1.532cm}}{\pgfqpoint{0.137cm}{1.532cm}}
\pgfpathcurveto{\pgfqpoint{0.1cm}{1.532cm}}{\pgfqpoint{0.066cm}{1.518cm}}{\pgfqpoint{0.04cm}{1.492cm}}
\pgfpathcurveto{\pgfqpoint{0.014cm}{1.467cm}}{\pgfqpoint{0cm}{1.432cm}}{\pgfqpoint{0cm}{1.395cm}}
\pgfpathcurveto{\pgfqpoint{0cm}{1.359cm}}{\pgfqpoint{0.014cm}{1.324cm}}{\pgfqpoint{0.04cm}{1.299cm}}
\pgfpathcurveto{\pgfqpoint{0.066cm}{1.273cm}}{\pgfqpoint{0.1cm}{1.258cm}}{\pgfqpoint{0.137cm}{1.258cm}}
\pgfpathcurveto{\pgfqpoint{0.173cm}{1.258cm}}{\pgfqpoint{0.207cm}{1.273cm}}{\pgfqpoint{0.233cm}{1.299cm}}
\pgfpathcurveto{\pgfqpoint{0.259cm}{1.324cm}}{\pgfqpoint{0.273cm}{1.359cm}}{\pgfqpoint{0.273cm}{1.395cm}}
\pgfusepath{fill}
\begin{pgfscope}
\pgfsetdash{}{0cm}
\pgfsetlinewidth{0.818mm}
\pgfsetmiterlimit{7.0}
\pgfpathmoveto{\pgfqpoint{0.682cm}{0.671cm}}
\pgfpathlineto{\pgfqpoint{0.679cm}{1.418cm}}
\pgfusepath{stroke}
\end{pgfscope}
\pgfpathmoveto{\pgfqpoint{0.815cm}{1.399cm}}
\pgfpathcurveto{\pgfqpoint{0.815cm}{1.435cm}}{\pgfqpoint{0.801cm}{1.47cm}}{\pgfqpoint{0.775cm}{1.496cm}}
\pgfpathcurveto{\pgfqpoint{0.75cm}{1.521cm}}{\pgfqpoint{0.715cm}{1.536cm}}{\pgfqpoint{0.679cm}{1.536cm}}
\pgfpathcurveto{\pgfqpoint{0.643cm}{1.536cm}}{\pgfqpoint{0.608cm}{1.521cm}}{\pgfqpoint{0.582cm}{1.496cm}}
\pgfpathcurveto{\pgfqpoint{0.557cm}{1.47cm}}{\pgfqpoint{0.542cm}{1.435cm}}{\pgfqpoint{0.542cm}{1.399cm}}
\pgfpathcurveto{\pgfqpoint{0.542cm}{1.363cm}}{\pgfqpoint{0.557cm}{1.328cm}}{\pgfqpoint{0.582cm}{1.302cm}}
\pgfpathcurveto{\pgfqpoint{0.608cm}{1.276cm}}{\pgfqpoint{0.643cm}{1.262cm}}{\pgfqpoint{0.679cm}{1.262cm}}
\pgfpathcurveto{\pgfqpoint{0.715cm}{1.262cm}}{\pgfqpoint{0.75cm}{1.276cm}}{\pgfqpoint{0.775cm}{1.302cm}}
\pgfpathcurveto{\pgfqpoint{0.801cm}{1.328cm}}{\pgfqpoint{0.815cm}{1.363cm}}{\pgfqpoint{0.815cm}{1.399cm}}
\pgfusepath{fill}
\pgfpathmoveto{\pgfqpoint{1.345cm}{1.371cm}}
\pgfpathcurveto{\pgfqpoint{1.345cm}{1.408cm}}{\pgfqpoint{1.331cm}{1.442cm}}{\pgfqpoint{1.305cm}{1.468cm}}
\pgfpathcurveto{\pgfqpoint{1.28cm}{1.494cm}}{\pgfqpoint{1.245cm}{1.508cm}}{\pgfqpoint{1.209cm}{1.508cm}}
\pgfpathcurveto{\pgfqpoint{1.172cm}{1.508cm}}{\pgfqpoint{1.138cm}{1.494cm}}{\pgfqpoint{1.112cm}{1.468cm}}
\pgfpathcurveto{\pgfqpoint{1.087cm}{1.442cm}}{\pgfqpoint{1.072cm}{1.408cm}}{\pgfqpoint{1.072cm}{1.371cm}}
\pgfpathcurveto{\pgfqpoint{1.072cm}{1.335cm}}{\pgfqpoint{1.087cm}{1.3cm}}{\pgfqpoint{1.112cm}{1.274cm}}
\pgfpathcurveto{\pgfqpoint{1.138cm}{1.249cm}}{\pgfqpoint{1.172cm}{1.234cm}}{\pgfqpoint{1.209cm}{1.234cm}}
\pgfpathcurveto{\pgfqpoint{1.245cm}{1.234cm}}{\pgfqpoint{1.28cm}{1.249cm}}{\pgfqpoint{1.305cm}{1.274cm}}
\pgfpathcurveto{\pgfqpoint{1.331cm}{1.3cm}}{\pgfqpoint{1.345cm}{1.335cm}}{\pgfqpoint{1.345cm}{1.371cm}}
\pgfusepath{fill}
\begin{pgfscope}
\pgfsetdash{}{0cm}
\pgfsetlinewidth{0.818mm}
\pgfsetroundcap
\pgfsetmiterlimit{4.0}
\pgfpathmoveto{\pgfqpoint{0.682cm}{0.671cm}}
\pgfpathlineto{\pgfqpoint{0.682cm}{0.042cm}}
\pgfusepath{stroke}
\end{pgfscope}
\end{pgfscope}
\end{pgfscope}
\end{pgfscope}
\end{tikzpicture}}} - 3
   b_{M, \varepsilon} X_{M, \varepsilon}, \]
where $b_{M, \varepsilon}, \tilde{b}_{M, \varepsilon} (t)$ are suitable
renormalization constants. It follows from standard estimates that $|
\tilde{b}_{M, \varepsilon} (t) - b_{M, \varepsilon} | \lesssim | \log t |$
uniformly in $M, \varepsilon$. We denote collectively
\begin{equation}\label{eq:XX}
 \mathbb{X}_{M, \varepsilon} \assign (X_{M, \varepsilon}, \llbracket X_{M,
   \varepsilon}^2 \rrbracket, X_{M, \varepsilon}^{\!\resizebox{0.6em}{!}{
\begin{tikzpicture}
\pgfpathmoveto{\pgfqpoint{0cm}{-0.035cm}}
\pgfpathlineto{\pgfqpoint{1.376cm}{-0.035cm}}
\pgfpathlineto{\pgfqpoint{1.376cm}{1.552cm}}
\pgfpathlineto{\pgfqpoint{0cm}{1.552cm}}
\pgfpathclose
\pgfusepath{clip}
\begin{pgfscope}
\begin{pgfscope}
\pgfpathmoveto{\pgfqpoint{0cm}{-0.035cm}}
\pgfpathlineto{\pgfqpoint{1.376cm}{-0.035cm}}
\pgfpathlineto{\pgfqpoint{1.376cm}{1.552cm}}
\pgfpathlineto{\pgfqpoint{0cm}{1.552cm}}
\pgfpathclose
\pgfusepath{clip}
\begin{pgfscope}
\begin{pgfscope}
\pgfsetdash{}{0cm}
\pgfsetlinewidth{0.818mm}
\pgfsetroundcap
\pgfsetroundjoin
\pgfsetmiterlimit{7.0}
\definecolor{eps2pgf_color}{gray}{0}\pgfsetstrokecolor{eps2pgf_color}\pgfsetfillcolor{eps2pgf_color}
\pgfpathmoveto{\pgfqpoint{0.117cm}{1.421cm}}
\pgfpathlineto{\pgfqpoint{0.682cm}{0.671cm}}
\pgfpathlineto{\pgfqpoint{1.246cm}{1.421cm}}
\pgfusepath{stroke}
\end{pgfscope}
\definecolor{eps2pgf_color}{gray}{0}\pgfsetstrokecolor{eps2pgf_color}\pgfsetfillcolor{eps2pgf_color}
\pgfpathmoveto{\pgfqpoint{0.273cm}{1.395cm}}
\pgfpathcurveto{\pgfqpoint{0.273cm}{1.432cm}}{\pgfqpoint{0.259cm}{1.467cm}}{\pgfqpoint{0.233cm}{1.492cm}}
\pgfpathcurveto{\pgfqpoint{0.207cm}{1.518cm}}{\pgfqpoint{0.173cm}{1.532cm}}{\pgfqpoint{0.137cm}{1.532cm}}
\pgfpathcurveto{\pgfqpoint{0.1cm}{1.532cm}}{\pgfqpoint{0.066cm}{1.518cm}}{\pgfqpoint{0.04cm}{1.492cm}}
\pgfpathcurveto{\pgfqpoint{0.014cm}{1.467cm}}{\pgfqpoint{0cm}{1.432cm}}{\pgfqpoint{0cm}{1.395cm}}
\pgfpathcurveto{\pgfqpoint{0cm}{1.359cm}}{\pgfqpoint{0.014cm}{1.324cm}}{\pgfqpoint{0.04cm}{1.299cm}}
\pgfpathcurveto{\pgfqpoint{0.066cm}{1.273cm}}{\pgfqpoint{0.1cm}{1.258cm}}{\pgfqpoint{0.137cm}{1.258cm}}
\pgfpathcurveto{\pgfqpoint{0.173cm}{1.258cm}}{\pgfqpoint{0.207cm}{1.273cm}}{\pgfqpoint{0.233cm}{1.299cm}}
\pgfpathcurveto{\pgfqpoint{0.259cm}{1.324cm}}{\pgfqpoint{0.273cm}{1.359cm}}{\pgfqpoint{0.273cm}{1.395cm}}
\pgfusepath{fill}
\begin{pgfscope}
\pgfsetdash{}{0cm}
\pgfsetlinewidth{0.818mm}
\pgfsetmiterlimit{7.0}
\pgfpathmoveto{\pgfqpoint{0.682cm}{0.671cm}}
\pgfpathlineto{\pgfqpoint{0.679cm}{1.418cm}}
\pgfusepath{stroke}
\end{pgfscope}
\pgfpathmoveto{\pgfqpoint{0.815cm}{1.399cm}}
\pgfpathcurveto{\pgfqpoint{0.815cm}{1.435cm}}{\pgfqpoint{0.801cm}{1.47cm}}{\pgfqpoint{0.775cm}{1.496cm}}
\pgfpathcurveto{\pgfqpoint{0.75cm}{1.521cm}}{\pgfqpoint{0.715cm}{1.536cm}}{\pgfqpoint{0.679cm}{1.536cm}}
\pgfpathcurveto{\pgfqpoint{0.643cm}{1.536cm}}{\pgfqpoint{0.608cm}{1.521cm}}{\pgfqpoint{0.582cm}{1.496cm}}
\pgfpathcurveto{\pgfqpoint{0.557cm}{1.47cm}}{\pgfqpoint{0.542cm}{1.435cm}}{\pgfqpoint{0.542cm}{1.399cm}}
\pgfpathcurveto{\pgfqpoint{0.542cm}{1.363cm}}{\pgfqpoint{0.557cm}{1.328cm}}{\pgfqpoint{0.582cm}{1.302cm}}
\pgfpathcurveto{\pgfqpoint{0.608cm}{1.276cm}}{\pgfqpoint{0.643cm}{1.262cm}}{\pgfqpoint{0.679cm}{1.262cm}}
\pgfpathcurveto{\pgfqpoint{0.715cm}{1.262cm}}{\pgfqpoint{0.75cm}{1.276cm}}{\pgfqpoint{0.775cm}{1.302cm}}
\pgfpathcurveto{\pgfqpoint{0.801cm}{1.328cm}}{\pgfqpoint{0.815cm}{1.363cm}}{\pgfqpoint{0.815cm}{1.399cm}}
\pgfusepath{fill}
\pgfpathmoveto{\pgfqpoint{1.345cm}{1.371cm}}
\pgfpathcurveto{\pgfqpoint{1.345cm}{1.408cm}}{\pgfqpoint{1.331cm}{1.442cm}}{\pgfqpoint{1.305cm}{1.468cm}}
\pgfpathcurveto{\pgfqpoint{1.28cm}{1.494cm}}{\pgfqpoint{1.245cm}{1.508cm}}{\pgfqpoint{1.209cm}{1.508cm}}
\pgfpathcurveto{\pgfqpoint{1.172cm}{1.508cm}}{\pgfqpoint{1.138cm}{1.494cm}}{\pgfqpoint{1.112cm}{1.468cm}}
\pgfpathcurveto{\pgfqpoint{1.087cm}{1.442cm}}{\pgfqpoint{1.072cm}{1.408cm}}{\pgfqpoint{1.072cm}{1.371cm}}
\pgfpathcurveto{\pgfqpoint{1.072cm}{1.335cm}}{\pgfqpoint{1.087cm}{1.3cm}}{\pgfqpoint{1.112cm}{1.274cm}}
\pgfpathcurveto{\pgfqpoint{1.138cm}{1.249cm}}{\pgfqpoint{1.172cm}{1.234cm}}{\pgfqpoint{1.209cm}{1.234cm}}
\pgfpathcurveto{\pgfqpoint{1.245cm}{1.234cm}}{\pgfqpoint{1.28cm}{1.249cm}}{\pgfqpoint{1.305cm}{1.274cm}}
\pgfpathcurveto{\pgfqpoint{1.331cm}{1.3cm}}{\pgfqpoint{1.345cm}{1.335cm}}{\pgfqpoint{1.345cm}{1.371cm}}
\pgfusepath{fill}
\begin{pgfscope}
\pgfsetdash{}{0cm}
\pgfsetlinewidth{0.818mm}
\pgfsetroundcap
\pgfsetmiterlimit{4.0}
\pgfpathmoveto{\pgfqpoint{0.682cm}{0.671cm}}
\pgfpathlineto{\pgfqpoint{0.682cm}{0.042cm}}
\pgfusepath{stroke}
\end{pgfscope}
\end{pgfscope}
\end{pgfscope}
\end{pgfscope}
\end{tikzpicture}}},
   X_{M, \varepsilon}^{\!\resizebox{!}{.8em}{
\begin{tikzpicture}
\pgfpathmoveto{\pgfqpoint{0cm}{-0.035cm}}
\pgfpathlineto{\pgfqpoint{1.976cm}{-0.035cm}}
\pgfpathlineto{\pgfqpoint{1.976cm}{1.94cm}}
\pgfpathlineto{\pgfqpoint{0cm}{1.94cm}}
\pgfpathclose
\pgfusepath{clip}
\begin{pgfscope}
\begin{pgfscope}
\pgfpathmoveto{\pgfqpoint{0cm}{-0.035cm}}
\pgfpathlineto{\pgfqpoint{1.976cm}{-0.035cm}}
\pgfpathlineto{\pgfqpoint{1.976cm}{1.94cm}}
\pgfpathlineto{\pgfqpoint{0cm}{1.94cm}}
\pgfpathclose
\pgfusepath{clip}
\begin{pgfscope}
\begin{pgfscope}
\pgfsetdash{}{0cm}
\pgfsetlinewidth{0.818mm}
\pgfsetroundcap
\pgfsetroundjoin
\pgfsetmiterlimit{7.0}
\definecolor{eps2pgf_color}{gray}{0}\pgfsetstrokecolor{eps2pgf_color}\pgfsetfillcolor{eps2pgf_color}
\pgfpathmoveto{\pgfqpoint{0.117cm}{1.815cm}}
\pgfpathlineto{\pgfqpoint{0.682cm}{1.065cm}}
\pgfpathlineto{\pgfqpoint{1.246cm}{1.815cm}}
\pgfusepath{stroke}
\end{pgfscope}
\definecolor{eps2pgf_color}{gray}{0}\pgfsetstrokecolor{eps2pgf_color}\pgfsetfillcolor{eps2pgf_color}
\pgfpathmoveto{\pgfqpoint{0.273cm}{1.789cm}}
\pgfpathcurveto{\pgfqpoint{0.273cm}{1.825cm}}{\pgfqpoint{0.259cm}{1.86cm}}{\pgfqpoint{0.233cm}{1.886cm}}
\pgfpathcurveto{\pgfqpoint{0.207cm}{1.912cm}}{\pgfqpoint{0.173cm}{1.926cm}}{\pgfqpoint{0.137cm}{1.926cm}}
\pgfpathcurveto{\pgfqpoint{0.1cm}{1.926cm}}{\pgfqpoint{0.066cm}{1.912cm}}{\pgfqpoint{0.04cm}{1.886cm}}
\pgfpathcurveto{\pgfqpoint{0.014cm}{1.86cm}}{\pgfqpoint{0cm}{1.825cm}}{\pgfqpoint{0cm}{1.789cm}}
\pgfpathcurveto{\pgfqpoint{0cm}{1.753cm}}{\pgfqpoint{0.014cm}{1.718cm}}{\pgfqpoint{0.04cm}{1.692cm}}
\pgfpathcurveto{\pgfqpoint{0.066cm}{1.667cm}}{\pgfqpoint{0.1cm}{1.652cm}}{\pgfqpoint{0.137cm}{1.652cm}}
\pgfpathcurveto{\pgfqpoint{0.173cm}{1.652cm}}{\pgfqpoint{0.207cm}{1.667cm}}{\pgfqpoint{0.233cm}{1.692cm}}
\pgfpathcurveto{\pgfqpoint{0.259cm}{1.718cm}}{\pgfqpoint{0.273cm}{1.753cm}}{\pgfqpoint{0.273cm}{1.789cm}}
\pgfusepath{fill}
\begin{pgfscope}
\pgfsetdash{}{0cm}
\pgfsetlinewidth{0.818mm}
\pgfsetmiterlimit{7.0}
\pgfpathmoveto{\pgfqpoint{0.682cm}{1.065cm}}
\pgfpathlineto{\pgfqpoint{0.679cm}{1.812cm}}
\pgfusepath{stroke}
\end{pgfscope}
\pgfpathmoveto{\pgfqpoint{0.815cm}{1.793cm}}
\pgfpathcurveto{\pgfqpoint{0.815cm}{1.829cm}}{\pgfqpoint{0.801cm}{1.864cm}}{\pgfqpoint{0.775cm}{1.89cm}}
\pgfpathcurveto{\pgfqpoint{0.75cm}{1.915cm}}{\pgfqpoint{0.715cm}{1.93cm}}{\pgfqpoint{0.679cm}{1.93cm}}
\pgfpathcurveto{\pgfqpoint{0.643cm}{1.93cm}}{\pgfqpoint{0.608cm}{1.915cm}}{\pgfqpoint{0.582cm}{1.89cm}}
\pgfpathcurveto{\pgfqpoint{0.557cm}{1.864cm}}{\pgfqpoint{0.542cm}{1.829cm}}{\pgfqpoint{0.542cm}{1.793cm}}
\pgfpathcurveto{\pgfqpoint{0.542cm}{1.756cm}}{\pgfqpoint{0.557cm}{1.722cm}}{\pgfqpoint{0.582cm}{1.696cm}}
\pgfpathcurveto{\pgfqpoint{0.608cm}{1.67cm}}{\pgfqpoint{0.643cm}{1.656cm}}{\pgfqpoint{0.679cm}{1.656cm}}
\pgfpathcurveto{\pgfqpoint{0.715cm}{1.656cm}}{\pgfqpoint{0.75cm}{1.67cm}}{\pgfqpoint{0.775cm}{1.696cm}}
\pgfpathcurveto{\pgfqpoint{0.801cm}{1.722cm}}{\pgfqpoint{0.815cm}{1.756cm}}{\pgfqpoint{0.815cm}{1.793cm}}
\pgfusepath{fill}
\pgfpathmoveto{\pgfqpoint{1.345cm}{1.765cm}}
\pgfpathcurveto{\pgfqpoint{1.345cm}{1.801cm}}{\pgfqpoint{1.331cm}{1.836cm}}{\pgfqpoint{1.305cm}{1.862cm}}
\pgfpathcurveto{\pgfqpoint{1.28cm}{1.887cm}}{\pgfqpoint{1.245cm}{1.902cm}}{\pgfqpoint{1.209cm}{1.902cm}}
\pgfpathcurveto{\pgfqpoint{1.172cm}{1.902cm}}{\pgfqpoint{1.138cm}{1.887cm}}{\pgfqpoint{1.112cm}{1.862cm}}
\pgfpathcurveto{\pgfqpoint{1.087cm}{1.836cm}}{\pgfqpoint{1.072cm}{1.801cm}}{\pgfqpoint{1.072cm}{1.765cm}}
\pgfpathcurveto{\pgfqpoint{1.072cm}{1.728cm}}{\pgfqpoint{1.087cm}{1.694cm}}{\pgfqpoint{1.112cm}{1.668cm}}
\pgfpathcurveto{\pgfqpoint{1.138cm}{1.642cm}}{\pgfqpoint{1.172cm}{1.628cm}}{\pgfqpoint{1.209cm}{1.628cm}}
\pgfpathcurveto{\pgfqpoint{1.245cm}{1.628cm}}{\pgfqpoint{1.28cm}{1.642cm}}{\pgfqpoint{1.305cm}{1.668cm}}
\pgfpathcurveto{\pgfqpoint{1.331cm}{1.694cm}}{\pgfqpoint{1.345cm}{1.728cm}}{\pgfqpoint{1.345cm}{1.765cm}}
\pgfusepath{fill}
\begin{pgfscope}
\pgfsetdash{}{0cm}
\pgfsetlinewidth{0.818mm}
\pgfsetroundcap
\pgfsetroundjoin
\pgfsetmiterlimit{7.0}
\pgfpathmoveto{\pgfqpoint{0.682cm}{1.065cm}}
\pgfpathlineto{\pgfqpoint{1.246cm}{0.315cm}}
\pgfpathlineto{\pgfqpoint{1.811cm}{1.065cm}}
\pgfusepath{stroke}
\end{pgfscope}
\pgfpathmoveto{\pgfqpoint{1.948cm}{1.065cm}}
\pgfpathcurveto{\pgfqpoint{1.948cm}{1.101cm}}{\pgfqpoint{1.933cm}{1.136cm}}{\pgfqpoint{1.907cm}{1.162cm}}
\pgfpathcurveto{\pgfqpoint{1.882cm}{1.187cm}}{\pgfqpoint{1.847cm}{1.202cm}}{\pgfqpoint{1.811cm}{1.202cm}}
\pgfpathcurveto{\pgfqpoint{1.775cm}{1.202cm}}{\pgfqpoint{1.74cm}{1.187cm}}{\pgfqpoint{1.714cm}{1.162cm}}
\pgfpathcurveto{\pgfqpoint{1.689cm}{1.136cm}}{\pgfqpoint{1.674cm}{1.101cm}}{\pgfqpoint{1.674cm}{1.065cm}}
\pgfpathcurveto{\pgfqpoint{1.674cm}{1.029cm}}{\pgfqpoint{1.689cm}{0.994cm}}{\pgfqpoint{1.714cm}{0.968cm}}
\pgfpathcurveto{\pgfqpoint{1.74cm}{0.942cm}}{\pgfqpoint{1.775cm}{0.928cm}}{\pgfqpoint{1.811cm}{0.928cm}}
\pgfpathcurveto{\pgfqpoint{1.847cm}{0.928cm}}{\pgfqpoint{1.882cm}{0.942cm}}{\pgfqpoint{1.907cm}{0.968cm}}
\pgfpathcurveto{\pgfqpoint{1.933cm}{0.994cm}}{\pgfqpoint{1.948cm}{1.029cm}}{\pgfqpoint{1.948cm}{1.065cm}}
\pgfusepath{fill}
\begin{pgfscope}
\pgfsetdash{}{0cm}
\pgfsetlinewidth{0.818mm}
\pgfsetmiterlimit{4.0}
\pgfpathmoveto{\pgfqpoint{1.383cm}{0.178cm}}
\pgfpathcurveto{\pgfqpoint{1.383cm}{0.214cm}}{\pgfqpoint{1.369cm}{0.249cm}}{\pgfqpoint{1.343cm}{0.275cm}}
\pgfpathcurveto{\pgfqpoint{1.317cm}{0.3cm}}{\pgfqpoint{1.283cm}{0.315cm}}{\pgfqpoint{1.246cm}{0.315cm}}
\pgfpathcurveto{\pgfqpoint{1.21cm}{0.315cm}}{\pgfqpoint{1.175cm}{0.3cm}}{\pgfqpoint{1.15cm}{0.275cm}}
\pgfpathcurveto{\pgfqpoint{1.124cm}{0.249cm}}{\pgfqpoint{1.11cm}{0.214cm}}{\pgfqpoint{1.11cm}{0.178cm}}
\pgfpathcurveto{\pgfqpoint{1.11cm}{0.141cm}}{\pgfqpoint{1.124cm}{0.107cm}}{\pgfqpoint{1.15cm}{0.081cm}}
\pgfpathcurveto{\pgfqpoint{1.175cm}{0.055cm}}{\pgfqpoint{1.21cm}{0.041cm}}{\pgfqpoint{1.246cm}{0.041cm}}
\pgfpathcurveto{\pgfqpoint{1.283cm}{0.041cm}}{\pgfqpoint{1.317cm}{0.055cm}}{\pgfqpoint{1.343cm}{0.081cm}}
\pgfpathcurveto{\pgfqpoint{1.369cm}{0.107cm}}{\pgfqpoint{1.383cm}{0.141cm}}{\pgfqpoint{1.383cm}{0.178cm}}
\pgfusepath{stroke}
\end{pgfscope}
\end{pgfscope}
\end{pgfscope}
\end{pgfscope}
\end{tikzpicture}}}, X_{M, \varepsilon}^{\!\resizebox{!}{.8em}{
\begin{tikzpicture}
\pgfpathmoveto{\pgfqpoint{0cm}{-0.035cm}}
\pgfpathlineto{\pgfqpoint{1.976cm}{-0.035cm}}
\pgfpathlineto{\pgfqpoint{1.976cm}{1.94cm}}
\pgfpathlineto{\pgfqpoint{0cm}{1.94cm}}
\pgfpathclose
\pgfusepath{clip}
\begin{pgfscope}
\begin{pgfscope}
\pgfpathmoveto{\pgfqpoint{0cm}{-0.035cm}}
\pgfpathlineto{\pgfqpoint{1.976cm}{-0.035cm}}
\pgfpathlineto{\pgfqpoint{1.976cm}{1.94cm}}
\pgfpathlineto{\pgfqpoint{0cm}{1.94cm}}
\pgfpathclose
\pgfusepath{clip}
\begin{pgfscope}
\begin{pgfscope}
\pgfsetdash{}{0cm}
\pgfsetlinewidth{0.818mm}
\pgfsetroundcap
\pgfsetroundjoin
\pgfsetmiterlimit{7.0}
\definecolor{eps2pgf_color}{gray}{0}\pgfsetstrokecolor{eps2pgf_color}\pgfsetfillcolor{eps2pgf_color}
\pgfpathmoveto{\pgfqpoint{0.117cm}{1.815cm}}
\pgfpathlineto{\pgfqpoint{0.682cm}{1.065cm}}
\pgfpathlineto{\pgfqpoint{1.246cm}{1.815cm}}
\pgfusepath{stroke}
\end{pgfscope}
\definecolor{eps2pgf_color}{gray}{0}\pgfsetstrokecolor{eps2pgf_color}\pgfsetfillcolor{eps2pgf_color}
\pgfpathmoveto{\pgfqpoint{0.273cm}{1.789cm}}
\pgfpathcurveto{\pgfqpoint{0.273cm}{1.825cm}}{\pgfqpoint{0.259cm}{1.86cm}}{\pgfqpoint{0.233cm}{1.886cm}}
\pgfpathcurveto{\pgfqpoint{0.207cm}{1.912cm}}{\pgfqpoint{0.173cm}{1.926cm}}{\pgfqpoint{0.137cm}{1.926cm}}
\pgfpathcurveto{\pgfqpoint{0.1cm}{1.926cm}}{\pgfqpoint{0.066cm}{1.912cm}}{\pgfqpoint{0.04cm}{1.886cm}}
\pgfpathcurveto{\pgfqpoint{0.014cm}{1.86cm}}{\pgfqpoint{0cm}{1.825cm}}{\pgfqpoint{0cm}{1.789cm}}
\pgfpathcurveto{\pgfqpoint{0cm}{1.753cm}}{\pgfqpoint{0.014cm}{1.718cm}}{\pgfqpoint{0.04cm}{1.692cm}}
\pgfpathcurveto{\pgfqpoint{0.066cm}{1.667cm}}{\pgfqpoint{0.1cm}{1.652cm}}{\pgfqpoint{0.137cm}{1.652cm}}
\pgfpathcurveto{\pgfqpoint{0.173cm}{1.652cm}}{\pgfqpoint{0.207cm}{1.667cm}}{\pgfqpoint{0.233cm}{1.692cm}}
\pgfpathcurveto{\pgfqpoint{0.259cm}{1.718cm}}{\pgfqpoint{0.273cm}{1.753cm}}{\pgfqpoint{0.273cm}{1.789cm}}
\pgfusepath{fill}
\pgfpathmoveto{\pgfqpoint{1.345cm}{1.765cm}}
\pgfpathcurveto{\pgfqpoint{1.345cm}{1.801cm}}{\pgfqpoint{1.331cm}{1.836cm}}{\pgfqpoint{1.305cm}{1.862cm}}
\pgfpathcurveto{\pgfqpoint{1.28cm}{1.887cm}}{\pgfqpoint{1.245cm}{1.902cm}}{\pgfqpoint{1.209cm}{1.902cm}}
\pgfpathcurveto{\pgfqpoint{1.172cm}{1.902cm}}{\pgfqpoint{1.138cm}{1.887cm}}{\pgfqpoint{1.112cm}{1.862cm}}
\pgfpathcurveto{\pgfqpoint{1.087cm}{1.836cm}}{\pgfqpoint{1.072cm}{1.801cm}}{\pgfqpoint{1.072cm}{1.765cm}}
\pgfpathcurveto{\pgfqpoint{1.072cm}{1.728cm}}{\pgfqpoint{1.087cm}{1.694cm}}{\pgfqpoint{1.112cm}{1.668cm}}
\pgfpathcurveto{\pgfqpoint{1.138cm}{1.642cm}}{\pgfqpoint{1.172cm}{1.628cm}}{\pgfqpoint{1.209cm}{1.628cm}}
\pgfpathcurveto{\pgfqpoint{1.245cm}{1.628cm}}{\pgfqpoint{1.28cm}{1.642cm}}{\pgfqpoint{1.305cm}{1.668cm}}
\pgfpathcurveto{\pgfqpoint{1.331cm}{1.694cm}}{\pgfqpoint{1.345cm}{1.728cm}}{\pgfqpoint{1.345cm}{1.765cm}}
\pgfusepath{fill}
\begin{pgfscope}
\pgfsetdash{}{0cm}
\pgfsetlinewidth{0.818mm}
\pgfsetroundcap
\pgfsetroundjoin
\pgfsetmiterlimit{7.0}
\pgfpathmoveto{\pgfqpoint{0.682cm}{1.065cm}}
\pgfpathlineto{\pgfqpoint{1.246cm}{0.315cm}}
\pgfpathlineto{\pgfqpoint{1.811cm}{1.065cm}}
\pgfusepath{stroke}
\end{pgfscope}
\pgfpathmoveto{\pgfqpoint{1.948cm}{1.065cm}}
\pgfpathcurveto{\pgfqpoint{1.948cm}{1.101cm}}{\pgfqpoint{1.933cm}{1.136cm}}{\pgfqpoint{1.907cm}{1.162cm}}
\pgfpathcurveto{\pgfqpoint{1.882cm}{1.187cm}}{\pgfqpoint{1.847cm}{1.202cm}}{\pgfqpoint{1.811cm}{1.202cm}}
\pgfpathcurveto{\pgfqpoint{1.775cm}{1.202cm}}{\pgfqpoint{1.74cm}{1.187cm}}{\pgfqpoint{1.714cm}{1.162cm}}
\pgfpathcurveto{\pgfqpoint{1.689cm}{1.136cm}}{\pgfqpoint{1.674cm}{1.101cm}}{\pgfqpoint{1.674cm}{1.065cm}}
\pgfpathcurveto{\pgfqpoint{1.674cm}{1.029cm}}{\pgfqpoint{1.689cm}{0.994cm}}{\pgfqpoint{1.714cm}{0.968cm}}
\pgfpathcurveto{\pgfqpoint{1.74cm}{0.942cm}}{\pgfqpoint{1.775cm}{0.928cm}}{\pgfqpoint{1.811cm}{0.928cm}}
\pgfpathcurveto{\pgfqpoint{1.847cm}{0.928cm}}{\pgfqpoint{1.882cm}{0.942cm}}{\pgfqpoint{1.907cm}{0.968cm}}
\pgfpathcurveto{\pgfqpoint{1.933cm}{0.994cm}}{\pgfqpoint{1.948cm}{1.029cm}}{\pgfqpoint{1.948cm}{1.065cm}}
\pgfusepath{fill}
\begin{pgfscope}
\pgfsetdash{}{0cm}
\pgfsetlinewidth{0.818mm}
\pgfsetmiterlimit{7.0}
\pgfpathmoveto{\pgfqpoint{1.246cm}{0.315cm}}
\pgfpathlineto{\pgfqpoint{1.244cm}{1.061cm}}
\pgfusepath{stroke}
\end{pgfscope}
\pgfpathmoveto{\pgfqpoint{1.38cm}{1.065cm}}
\pgfpathcurveto{\pgfqpoint{1.38cm}{1.101cm}}{\pgfqpoint{1.366cm}{1.136cm}}{\pgfqpoint{1.34cm}{1.162cm}}
\pgfpathcurveto{\pgfqpoint{1.315cm}{1.187cm}}{\pgfqpoint{1.28cm}{1.202cm}}{\pgfqpoint{1.244cm}{1.202cm}}
\pgfpathcurveto{\pgfqpoint{1.207cm}{1.202cm}}{\pgfqpoint{1.173cm}{1.187cm}}{\pgfqpoint{1.147cm}{1.162cm}}
\pgfpathcurveto{\pgfqpoint{1.121cm}{1.136cm}}{\pgfqpoint{1.107cm}{1.101cm}}{\pgfqpoint{1.107cm}{1.065cm}}
\pgfpathcurveto{\pgfqpoint{1.107cm}{1.029cm}}{\pgfqpoint{1.121cm}{0.994cm}}{\pgfqpoint{1.147cm}{0.968cm}}
\pgfpathcurveto{\pgfqpoint{1.173cm}{0.942cm}}{\pgfqpoint{1.207cm}{0.928cm}}{\pgfqpoint{1.244cm}{0.928cm}}
\pgfpathcurveto{\pgfqpoint{1.28cm}{0.928cm}}{\pgfqpoint{1.315cm}{0.942cm}}{\pgfqpoint{1.34cm}{0.968cm}}
\pgfpathcurveto{\pgfqpoint{1.366cm}{0.994cm}}{\pgfqpoint{1.38cm}{1.029cm}}{\pgfqpoint{1.38cm}{1.065cm}}
\pgfusepath{fill}
\begin{pgfscope}
\pgfsetdash{}{0cm}
\pgfsetlinewidth{0.818mm}
\pgfsetmiterlimit{4.0}
\pgfpathmoveto{\pgfqpoint{1.383cm}{0.178cm}}
\pgfpathcurveto{\pgfqpoint{1.383cm}{0.214cm}}{\pgfqpoint{1.369cm}{0.249cm}}{\pgfqpoint{1.343cm}{0.275cm}}
\pgfpathcurveto{\pgfqpoint{1.317cm}{0.3cm}}{\pgfqpoint{1.283cm}{0.315cm}}{\pgfqpoint{1.246cm}{0.315cm}}
\pgfpathcurveto{\pgfqpoint{1.21cm}{0.315cm}}{\pgfqpoint{1.175cm}{0.3cm}}{\pgfqpoint{1.15cm}{0.275cm}}
\pgfpathcurveto{\pgfqpoint{1.124cm}{0.249cm}}{\pgfqpoint{1.11cm}{0.214cm}}{\pgfqpoint{1.11cm}{0.178cm}}
\pgfpathcurveto{\pgfqpoint{1.11cm}{0.141cm}}{\pgfqpoint{1.124cm}{0.107cm}}{\pgfqpoint{1.15cm}{0.081cm}}
\pgfpathcurveto{\pgfqpoint{1.175cm}{0.055cm}}{\pgfqpoint{1.21cm}{0.041cm}}{\pgfqpoint{1.246cm}{0.041cm}}
\pgfpathcurveto{\pgfqpoint{1.283cm}{0.041cm}}{\pgfqpoint{1.317cm}{0.055cm}}{\pgfqpoint{1.343cm}{0.081cm}}
\pgfpathcurveto{\pgfqpoint{1.369cm}{0.107cm}}{\pgfqpoint{1.383cm}{0.141cm}}{\pgfqpoint{1.383cm}{0.178cm}}
\pgfusepath{stroke}
\end{pgfscope}
\end{pgfscope}
\end{pgfscope}
\end{pgfscope}
\end{tikzpicture}}},
   \tilde{X}_{M, \varepsilon}^{\!\resizebox{!}{.8em}{
\begin{tikzpicture}
\pgfpathmoveto{\pgfqpoint{0cm}{-0.035cm}}
\pgfpathlineto{\pgfqpoint{1.976cm}{-0.035cm}}
\pgfpathlineto{\pgfqpoint{1.976cm}{1.94cm}}
\pgfpathlineto{\pgfqpoint{0cm}{1.94cm}}
\pgfpathclose
\pgfusepath{clip}
\begin{pgfscope}
\begin{pgfscope}
\pgfpathmoveto{\pgfqpoint{0cm}{-0.035cm}}
\pgfpathlineto{\pgfqpoint{1.976cm}{-0.035cm}}
\pgfpathlineto{\pgfqpoint{1.976cm}{1.94cm}}
\pgfpathlineto{\pgfqpoint{0cm}{1.94cm}}
\pgfpathclose
\pgfusepath{clip}
\begin{pgfscope}
\begin{pgfscope}
\pgfsetdash{}{0cm}
\pgfsetlinewidth{0.818mm}
\pgfsetroundcap
\pgfsetroundjoin
\pgfsetmiterlimit{7.0}
\definecolor{eps2pgf_color}{gray}{0}\pgfsetstrokecolor{eps2pgf_color}\pgfsetfillcolor{eps2pgf_color}
\pgfpathmoveto{\pgfqpoint{0.117cm}{1.815cm}}
\pgfpathlineto{\pgfqpoint{0.682cm}{1.065cm}}
\pgfpathlineto{\pgfqpoint{1.246cm}{1.815cm}}
\pgfusepath{stroke}
\end{pgfscope}
\definecolor{eps2pgf_color}{gray}{0}\pgfsetstrokecolor{eps2pgf_color}\pgfsetfillcolor{eps2pgf_color}
\pgfpathmoveto{\pgfqpoint{0.273cm}{1.789cm}}
\pgfpathcurveto{\pgfqpoint{0.273cm}{1.825cm}}{\pgfqpoint{0.259cm}{1.86cm}}{\pgfqpoint{0.233cm}{1.886cm}}
\pgfpathcurveto{\pgfqpoint{0.207cm}{1.912cm}}{\pgfqpoint{0.173cm}{1.926cm}}{\pgfqpoint{0.137cm}{1.926cm}}
\pgfpathcurveto{\pgfqpoint{0.1cm}{1.926cm}}{\pgfqpoint{0.066cm}{1.912cm}}{\pgfqpoint{0.04cm}{1.886cm}}
\pgfpathcurveto{\pgfqpoint{0.014cm}{1.86cm}}{\pgfqpoint{0cm}{1.825cm}}{\pgfqpoint{0cm}{1.789cm}}
\pgfpathcurveto{\pgfqpoint{0cm}{1.753cm}}{\pgfqpoint{0.014cm}{1.718cm}}{\pgfqpoint{0.04cm}{1.692cm}}
\pgfpathcurveto{\pgfqpoint{0.066cm}{1.667cm}}{\pgfqpoint{0.1cm}{1.652cm}}{\pgfqpoint{0.137cm}{1.652cm}}
\pgfpathcurveto{\pgfqpoint{0.173cm}{1.652cm}}{\pgfqpoint{0.207cm}{1.667cm}}{\pgfqpoint{0.233cm}{1.692cm}}
\pgfpathcurveto{\pgfqpoint{0.259cm}{1.718cm}}{\pgfqpoint{0.273cm}{1.753cm}}{\pgfqpoint{0.273cm}{1.789cm}}
\pgfusepath{fill}
\pgfpathmoveto{\pgfqpoint{1.345cm}{1.765cm}}
\pgfpathcurveto{\pgfqpoint{1.345cm}{1.801cm}}{\pgfqpoint{1.331cm}{1.836cm}}{\pgfqpoint{1.305cm}{1.862cm}}
\pgfpathcurveto{\pgfqpoint{1.28cm}{1.887cm}}{\pgfqpoint{1.245cm}{1.902cm}}{\pgfqpoint{1.209cm}{1.902cm}}
\pgfpathcurveto{\pgfqpoint{1.172cm}{1.902cm}}{\pgfqpoint{1.138cm}{1.887cm}}{\pgfqpoint{1.112cm}{1.862cm}}
\pgfpathcurveto{\pgfqpoint{1.087cm}{1.836cm}}{\pgfqpoint{1.072cm}{1.801cm}}{\pgfqpoint{1.072cm}{1.765cm}}
\pgfpathcurveto{\pgfqpoint{1.072cm}{1.728cm}}{\pgfqpoint{1.087cm}{1.694cm}}{\pgfqpoint{1.112cm}{1.668cm}}
\pgfpathcurveto{\pgfqpoint{1.138cm}{1.642cm}}{\pgfqpoint{1.172cm}{1.628cm}}{\pgfqpoint{1.209cm}{1.628cm}}
\pgfpathcurveto{\pgfqpoint{1.245cm}{1.628cm}}{\pgfqpoint{1.28cm}{1.642cm}}{\pgfqpoint{1.305cm}{1.668cm}}
\pgfpathcurveto{\pgfqpoint{1.331cm}{1.694cm}}{\pgfqpoint{1.345cm}{1.728cm}}{\pgfqpoint{1.345cm}{1.765cm}}
\pgfusepath{fill}
\begin{pgfscope}
\pgfsetdash{}{0cm}
\pgfsetlinewidth{0.818mm}
\pgfsetroundcap
\pgfsetroundjoin
\pgfsetmiterlimit{7.0}
\pgfpathmoveto{\pgfqpoint{0.682cm}{1.065cm}}
\pgfpathlineto{\pgfqpoint{1.246cm}{0.315cm}}
\pgfpathlineto{\pgfqpoint{1.811cm}{1.065cm}}
\pgfusepath{stroke}
\end{pgfscope}
\pgfpathmoveto{\pgfqpoint{1.948cm}{1.065cm}}
\pgfpathcurveto{\pgfqpoint{1.948cm}{1.101cm}}{\pgfqpoint{1.933cm}{1.136cm}}{\pgfqpoint{1.907cm}{1.162cm}}
\pgfpathcurveto{\pgfqpoint{1.882cm}{1.187cm}}{\pgfqpoint{1.847cm}{1.202cm}}{\pgfqpoint{1.811cm}{1.202cm}}
\pgfpathcurveto{\pgfqpoint{1.775cm}{1.202cm}}{\pgfqpoint{1.74cm}{1.187cm}}{\pgfqpoint{1.714cm}{1.162cm}}
\pgfpathcurveto{\pgfqpoint{1.689cm}{1.136cm}}{\pgfqpoint{1.674cm}{1.101cm}}{\pgfqpoint{1.674cm}{1.065cm}}
\pgfpathcurveto{\pgfqpoint{1.674cm}{1.029cm}}{\pgfqpoint{1.689cm}{0.994cm}}{\pgfqpoint{1.714cm}{0.968cm}}
\pgfpathcurveto{\pgfqpoint{1.74cm}{0.942cm}}{\pgfqpoint{1.775cm}{0.928cm}}{\pgfqpoint{1.811cm}{0.928cm}}
\pgfpathcurveto{\pgfqpoint{1.847cm}{0.928cm}}{\pgfqpoint{1.882cm}{0.942cm}}{\pgfqpoint{1.907cm}{0.968cm}}
\pgfpathcurveto{\pgfqpoint{1.933cm}{0.994cm}}{\pgfqpoint{1.948cm}{1.029cm}}{\pgfqpoint{1.948cm}{1.065cm}}
\pgfusepath{fill}
\begin{pgfscope}
\pgfsetdash{}{0cm}
\pgfsetlinewidth{0.818mm}
\pgfsetmiterlimit{7.0}
\pgfpathmoveto{\pgfqpoint{1.246cm}{0.315cm}}
\pgfpathlineto{\pgfqpoint{1.244cm}{1.061cm}}
\pgfusepath{stroke}
\end{pgfscope}
\pgfpathmoveto{\pgfqpoint{1.38cm}{1.065cm}}
\pgfpathcurveto{\pgfqpoint{1.38cm}{1.101cm}}{\pgfqpoint{1.366cm}{1.136cm}}{\pgfqpoint{1.34cm}{1.162cm}}
\pgfpathcurveto{\pgfqpoint{1.315cm}{1.187cm}}{\pgfqpoint{1.28cm}{1.202cm}}{\pgfqpoint{1.244cm}{1.202cm}}
\pgfpathcurveto{\pgfqpoint{1.207cm}{1.202cm}}{\pgfqpoint{1.173cm}{1.187cm}}{\pgfqpoint{1.147cm}{1.162cm}}
\pgfpathcurveto{\pgfqpoint{1.121cm}{1.136cm}}{\pgfqpoint{1.107cm}{1.101cm}}{\pgfqpoint{1.107cm}{1.065cm}}
\pgfpathcurveto{\pgfqpoint{1.107cm}{1.029cm}}{\pgfqpoint{1.121cm}{0.994cm}}{\pgfqpoint{1.147cm}{0.968cm}}
\pgfpathcurveto{\pgfqpoint{1.173cm}{0.942cm}}{\pgfqpoint{1.207cm}{0.928cm}}{\pgfqpoint{1.244cm}{0.928cm}}
\pgfpathcurveto{\pgfqpoint{1.28cm}{0.928cm}}{\pgfqpoint{1.315cm}{0.942cm}}{\pgfqpoint{1.34cm}{0.968cm}}
\pgfpathcurveto{\pgfqpoint{1.366cm}{0.994cm}}{\pgfqpoint{1.38cm}{1.029cm}}{\pgfqpoint{1.38cm}{1.065cm}}
\pgfusepath{fill}
\begin{pgfscope}
\pgfsetdash{}{0cm}
\pgfsetlinewidth{0.818mm}
\pgfsetmiterlimit{4.0}
\pgfpathmoveto{\pgfqpoint{1.383cm}{0.178cm}}
\pgfpathcurveto{\pgfqpoint{1.383cm}{0.214cm}}{\pgfqpoint{1.369cm}{0.249cm}}{\pgfqpoint{1.343cm}{0.275cm}}
\pgfpathcurveto{\pgfqpoint{1.317cm}{0.3cm}}{\pgfqpoint{1.283cm}{0.315cm}}{\pgfqpoint{1.246cm}{0.315cm}}
\pgfpathcurveto{\pgfqpoint{1.21cm}{0.315cm}}{\pgfqpoint{1.175cm}{0.3cm}}{\pgfqpoint{1.15cm}{0.275cm}}
\pgfpathcurveto{\pgfqpoint{1.124cm}{0.249cm}}{\pgfqpoint{1.11cm}{0.214cm}}{\pgfqpoint{1.11cm}{0.178cm}}
\pgfpathcurveto{\pgfqpoint{1.11cm}{0.141cm}}{\pgfqpoint{1.124cm}{0.107cm}}{\pgfqpoint{1.15cm}{0.081cm}}
\pgfpathcurveto{\pgfqpoint{1.175cm}{0.055cm}}{\pgfqpoint{1.21cm}{0.041cm}}{\pgfqpoint{1.246cm}{0.041cm}}
\pgfpathcurveto{\pgfqpoint{1.283cm}{0.041cm}}{\pgfqpoint{1.317cm}{0.055cm}}{\pgfqpoint{1.343cm}{0.081cm}}
\pgfpathcurveto{\pgfqpoint{1.369cm}{0.107cm}}{\pgfqpoint{1.383cm}{0.141cm}}{\pgfqpoint{1.383cm}{0.178cm}}
\pgfusepath{stroke}
\end{pgfscope}
\end{pgfscope}
\end{pgfscope}
\end{pgfscope}
\end{tikzpicture}}}, X_{M, \varepsilon}^{\!\resizebox{!}{.8em}{
\begin{tikzpicture}
\pgfpathmoveto{\pgfqpoint{0cm}{-0.035cm}}
\pgfpathlineto{\pgfqpoint{1.976cm}{-0.035cm}}
\pgfpathlineto{\pgfqpoint{1.976cm}{1.94cm}}
\pgfpathlineto{\pgfqpoint{0cm}{1.94cm}}
\pgfpathclose
\pgfusepath{clip}
\begin{pgfscope}
\begin{pgfscope}
\pgfpathmoveto{\pgfqpoint{0cm}{-0.035cm}}
\pgfpathlineto{\pgfqpoint{1.976cm}{-0.035cm}}
\pgfpathlineto{\pgfqpoint{1.976cm}{1.94cm}}
\pgfpathlineto{\pgfqpoint{0cm}{1.94cm}}
\pgfpathclose
\pgfusepath{clip}
\begin{pgfscope}
\begin{pgfscope}
\pgfsetdash{}{0cm}
\pgfsetlinewidth{0.818mm}
\pgfsetroundcap
\pgfsetroundjoin
\pgfsetmiterlimit{7.0}
\definecolor{eps2pgf_color}{gray}{0}\pgfsetstrokecolor{eps2pgf_color}\pgfsetfillcolor{eps2pgf_color}
\pgfpathmoveto{\pgfqpoint{0.117cm}{1.815cm}}
\pgfpathlineto{\pgfqpoint{0.682cm}{1.065cm}}
\pgfpathlineto{\pgfqpoint{1.246cm}{1.815cm}}
\pgfusepath{stroke}
\end{pgfscope}
\definecolor{eps2pgf_color}{gray}{0}\pgfsetstrokecolor{eps2pgf_color}\pgfsetfillcolor{eps2pgf_color}
\pgfpathmoveto{\pgfqpoint{0.273cm}{1.789cm}}
\pgfpathcurveto{\pgfqpoint{0.273cm}{1.825cm}}{\pgfqpoint{0.259cm}{1.86cm}}{\pgfqpoint{0.233cm}{1.886cm}}
\pgfpathcurveto{\pgfqpoint{0.207cm}{1.912cm}}{\pgfqpoint{0.173cm}{1.926cm}}{\pgfqpoint{0.137cm}{1.926cm}}
\pgfpathcurveto{\pgfqpoint{0.1cm}{1.926cm}}{\pgfqpoint{0.066cm}{1.912cm}}{\pgfqpoint{0.04cm}{1.886cm}}
\pgfpathcurveto{\pgfqpoint{0.014cm}{1.86cm}}{\pgfqpoint{0cm}{1.825cm}}{\pgfqpoint{0cm}{1.789cm}}
\pgfpathcurveto{\pgfqpoint{0cm}{1.753cm}}{\pgfqpoint{0.014cm}{1.718cm}}{\pgfqpoint{0.04cm}{1.692cm}}
\pgfpathcurveto{\pgfqpoint{0.066cm}{1.667cm}}{\pgfqpoint{0.1cm}{1.652cm}}{\pgfqpoint{0.137cm}{1.652cm}}
\pgfpathcurveto{\pgfqpoint{0.173cm}{1.652cm}}{\pgfqpoint{0.207cm}{1.667cm}}{\pgfqpoint{0.233cm}{1.692cm}}
\pgfpathcurveto{\pgfqpoint{0.259cm}{1.718cm}}{\pgfqpoint{0.273cm}{1.753cm}}{\pgfqpoint{0.273cm}{1.789cm}}
\pgfusepath{fill}
\begin{pgfscope}
\pgfsetdash{}{0cm}
\pgfsetlinewidth{0.818mm}
\pgfsetmiterlimit{7.0}
\pgfpathmoveto{\pgfqpoint{0.682cm}{1.065cm}}
\pgfpathlineto{\pgfqpoint{0.679cm}{1.812cm}}
\pgfusepath{stroke}
\end{pgfscope}
\pgfpathmoveto{\pgfqpoint{0.815cm}{1.793cm}}
\pgfpathcurveto{\pgfqpoint{0.815cm}{1.829cm}}{\pgfqpoint{0.801cm}{1.864cm}}{\pgfqpoint{0.775cm}{1.89cm}}
\pgfpathcurveto{\pgfqpoint{0.75cm}{1.915cm}}{\pgfqpoint{0.715cm}{1.93cm}}{\pgfqpoint{0.679cm}{1.93cm}}
\pgfpathcurveto{\pgfqpoint{0.643cm}{1.93cm}}{\pgfqpoint{0.608cm}{1.915cm}}{\pgfqpoint{0.582cm}{1.89cm}}
\pgfpathcurveto{\pgfqpoint{0.557cm}{1.864cm}}{\pgfqpoint{0.542cm}{1.829cm}}{\pgfqpoint{0.542cm}{1.793cm}}
\pgfpathcurveto{\pgfqpoint{0.542cm}{1.756cm}}{\pgfqpoint{0.557cm}{1.722cm}}{\pgfqpoint{0.582cm}{1.696cm}}
\pgfpathcurveto{\pgfqpoint{0.608cm}{1.67cm}}{\pgfqpoint{0.643cm}{1.656cm}}{\pgfqpoint{0.679cm}{1.656cm}}
\pgfpathcurveto{\pgfqpoint{0.715cm}{1.656cm}}{\pgfqpoint{0.75cm}{1.67cm}}{\pgfqpoint{0.775cm}{1.696cm}}
\pgfpathcurveto{\pgfqpoint{0.801cm}{1.722cm}}{\pgfqpoint{0.815cm}{1.756cm}}{\pgfqpoint{0.815cm}{1.793cm}}
\pgfusepath{fill}
\pgfpathmoveto{\pgfqpoint{1.345cm}{1.765cm}}
\pgfpathcurveto{\pgfqpoint{1.345cm}{1.801cm}}{\pgfqpoint{1.331cm}{1.836cm}}{\pgfqpoint{1.305cm}{1.862cm}}
\pgfpathcurveto{\pgfqpoint{1.28cm}{1.887cm}}{\pgfqpoint{1.245cm}{1.902cm}}{\pgfqpoint{1.209cm}{1.902cm}}
\pgfpathcurveto{\pgfqpoint{1.172cm}{1.902cm}}{\pgfqpoint{1.138cm}{1.887cm}}{\pgfqpoint{1.112cm}{1.862cm}}
\pgfpathcurveto{\pgfqpoint{1.087cm}{1.836cm}}{\pgfqpoint{1.072cm}{1.801cm}}{\pgfqpoint{1.072cm}{1.765cm}}
\pgfpathcurveto{\pgfqpoint{1.072cm}{1.728cm}}{\pgfqpoint{1.087cm}{1.694cm}}{\pgfqpoint{1.112cm}{1.668cm}}
\pgfpathcurveto{\pgfqpoint{1.138cm}{1.642cm}}{\pgfqpoint{1.172cm}{1.628cm}}{\pgfqpoint{1.209cm}{1.628cm}}
\pgfpathcurveto{\pgfqpoint{1.245cm}{1.628cm}}{\pgfqpoint{1.28cm}{1.642cm}}{\pgfqpoint{1.305cm}{1.668cm}}
\pgfpathcurveto{\pgfqpoint{1.331cm}{1.694cm}}{\pgfqpoint{1.345cm}{1.728cm}}{\pgfqpoint{1.345cm}{1.765cm}}
\pgfusepath{fill}
\begin{pgfscope}
\pgfsetdash{}{0cm}
\pgfsetlinewidth{0.818mm}
\pgfsetroundcap
\pgfsetroundjoin
\pgfsetmiterlimit{7.0}
\pgfpathmoveto{\pgfqpoint{0.682cm}{1.065cm}}
\pgfpathlineto{\pgfqpoint{1.246cm}{0.315cm}}
\pgfpathlineto{\pgfqpoint{1.811cm}{1.065cm}}
\pgfusepath{stroke}
\end{pgfscope}
\pgfpathmoveto{\pgfqpoint{1.948cm}{1.065cm}}
\pgfpathcurveto{\pgfqpoint{1.948cm}{1.101cm}}{\pgfqpoint{1.933cm}{1.136cm}}{\pgfqpoint{1.907cm}{1.162cm}}
\pgfpathcurveto{\pgfqpoint{1.882cm}{1.187cm}}{\pgfqpoint{1.847cm}{1.202cm}}{\pgfqpoint{1.811cm}{1.202cm}}
\pgfpathcurveto{\pgfqpoint{1.775cm}{1.202cm}}{\pgfqpoint{1.74cm}{1.187cm}}{\pgfqpoint{1.714cm}{1.162cm}}
\pgfpathcurveto{\pgfqpoint{1.689cm}{1.136cm}}{\pgfqpoint{1.674cm}{1.101cm}}{\pgfqpoint{1.674cm}{1.065cm}}
\pgfpathcurveto{\pgfqpoint{1.674cm}{1.029cm}}{\pgfqpoint{1.689cm}{0.994cm}}{\pgfqpoint{1.714cm}{0.968cm}}
\pgfpathcurveto{\pgfqpoint{1.74cm}{0.942cm}}{\pgfqpoint{1.775cm}{0.928cm}}{\pgfqpoint{1.811cm}{0.928cm}}
\pgfpathcurveto{\pgfqpoint{1.847cm}{0.928cm}}{\pgfqpoint{1.882cm}{0.942cm}}{\pgfqpoint{1.907cm}{0.968cm}}
\pgfpathcurveto{\pgfqpoint{1.933cm}{0.994cm}}{\pgfqpoint{1.948cm}{1.029cm}}{\pgfqpoint{1.948cm}{1.065cm}}
\pgfusepath{fill}
\begin{pgfscope}
\pgfsetdash{}{0cm}
\pgfsetlinewidth{0.818mm}
\pgfsetmiterlimit{7.0}
\pgfpathmoveto{\pgfqpoint{1.246cm}{0.315cm}}
\pgfpathlineto{\pgfqpoint{1.244cm}{1.061cm}}
\pgfusepath{stroke}
\end{pgfscope}
\pgfpathmoveto{\pgfqpoint{1.38cm}{1.065cm}}
\pgfpathcurveto{\pgfqpoint{1.38cm}{1.101cm}}{\pgfqpoint{1.366cm}{1.136cm}}{\pgfqpoint{1.34cm}{1.162cm}}
\pgfpathcurveto{\pgfqpoint{1.315cm}{1.187cm}}{\pgfqpoint{1.28cm}{1.202cm}}{\pgfqpoint{1.244cm}{1.202cm}}
\pgfpathcurveto{\pgfqpoint{1.207cm}{1.202cm}}{\pgfqpoint{1.173cm}{1.187cm}}{\pgfqpoint{1.147cm}{1.162cm}}
\pgfpathcurveto{\pgfqpoint{1.121cm}{1.136cm}}{\pgfqpoint{1.107cm}{1.101cm}}{\pgfqpoint{1.107cm}{1.065cm}}
\pgfpathcurveto{\pgfqpoint{1.107cm}{1.029cm}}{\pgfqpoint{1.121cm}{0.994cm}}{\pgfqpoint{1.147cm}{0.968cm}}
\pgfpathcurveto{\pgfqpoint{1.173cm}{0.942cm}}{\pgfqpoint{1.207cm}{0.928cm}}{\pgfqpoint{1.244cm}{0.928cm}}
\pgfpathcurveto{\pgfqpoint{1.28cm}{0.928cm}}{\pgfqpoint{1.315cm}{0.942cm}}{\pgfqpoint{1.34cm}{0.968cm}}
\pgfpathcurveto{\pgfqpoint{1.366cm}{0.994cm}}{\pgfqpoint{1.38cm}{1.029cm}}{\pgfqpoint{1.38cm}{1.065cm}}
\pgfusepath{fill}
\begin{pgfscope}
\pgfsetdash{}{0cm}
\pgfsetlinewidth{0.818mm}
\pgfsetmiterlimit{4.0}
\pgfpathmoveto{\pgfqpoint{1.383cm}{0.178cm}}
\pgfpathcurveto{\pgfqpoint{1.383cm}{0.214cm}}{\pgfqpoint{1.369cm}{0.249cm}}{\pgfqpoint{1.343cm}{0.275cm}}
\pgfpathcurveto{\pgfqpoint{1.317cm}{0.3cm}}{\pgfqpoint{1.283cm}{0.315cm}}{\pgfqpoint{1.246cm}{0.315cm}}
\pgfpathcurveto{\pgfqpoint{1.21cm}{0.315cm}}{\pgfqpoint{1.175cm}{0.3cm}}{\pgfqpoint{1.15cm}{0.275cm}}
\pgfpathcurveto{\pgfqpoint{1.124cm}{0.249cm}}{\pgfqpoint{1.11cm}{0.214cm}}{\pgfqpoint{1.11cm}{0.178cm}}
\pgfpathcurveto{\pgfqpoint{1.11cm}{0.141cm}}{\pgfqpoint{1.124cm}{0.107cm}}{\pgfqpoint{1.15cm}{0.081cm}}
\pgfpathcurveto{\pgfqpoint{1.175cm}{0.055cm}}{\pgfqpoint{1.21cm}{0.041cm}}{\pgfqpoint{1.246cm}{0.041cm}}
\pgfpathcurveto{\pgfqpoint{1.283cm}{0.041cm}}{\pgfqpoint{1.317cm}{0.055cm}}{\pgfqpoint{1.343cm}{0.081cm}}
\pgfpathcurveto{\pgfqpoint{1.369cm}{0.107cm}}{\pgfqpoint{1.383cm}{0.141cm}}{\pgfqpoint{1.383cm}{0.178cm}}
\pgfusepath{stroke}
\end{pgfscope}
\end{pgfscope}
\end{pgfscope}
\end{pgfscope}
\end{tikzpicture}}})
   .
   \end{equation}
   These objects can be constructed similarly as the usual $\Phi^{4}_{3}$ terms, see e.g. \cite{GH18,hairer_regularity_2015,mourrat_construction_2016}.
Note that we do not include $X_{M, \varepsilon}^{\!\resizebox{0.6em}{!}{
\begin{tikzpicture}
\pgfpathmoveto{\pgfqpoint{0cm}{0cm}}
\pgfpathlineto{\pgfqpoint{1.376cm}{0cm}}
\pgfpathlineto{\pgfqpoint{1.376cm}{1.588cm}}
\pgfpathlineto{\pgfqpoint{0cm}{1.588cm}}
\pgfpathclose
\pgfusepath{clip}
\begin{pgfscope}
\begin{pgfscope}
\pgfpathmoveto{\pgfqpoint{0cm}{0cm}}
\pgfpathlineto{\pgfqpoint{1.376cm}{0cm}}
\pgfpathlineto{\pgfqpoint{1.376cm}{1.588cm}}
\pgfpathlineto{\pgfqpoint{0cm}{1.588cm}}
\pgfpathclose
\pgfusepath{clip}
\begin{pgfscope}
\begin{pgfscope}
\definecolor{eps2pgf_color}{gray}{0.976471}\pgfsetstrokecolor{eps2pgf_color}\pgfsetfillcolor{eps2pgf_color}
\pgfpathmoveto{\pgfqpoint{0cm}{0cm}}
\pgfpathlineto{\pgfqpoint{1.376cm}{0cm}}
\pgfpathlineto{\pgfqpoint{1.376cm}{1.588cm}}
\pgfpathlineto{\pgfqpoint{0cm}{1.588cm}}
\pgfpathclose
\pgfusepath{fill}
\end{pgfscope}
\begin{pgfscope}
\pgfsetdash{}{0cm}
\pgfsetlinewidth{0.818mm}
\pgfsetroundcap
\pgfsetroundjoin
\pgfsetmiterlimit{7.0}
\definecolor{eps2pgf_color}{gray}{0}\pgfsetstrokecolor{eps2pgf_color}\pgfsetfillcolor{eps2pgf_color}
\pgfpathmoveto{\pgfqpoint{0.117cm}{1.476cm}}
\pgfpathlineto{\pgfqpoint{0.682cm}{0.726cm}}
\pgfpathlineto{\pgfqpoint{1.246cm}{1.476cm}}
\pgfusepath{stroke}
\end{pgfscope}
\definecolor{eps2pgf_color}{gray}{0}\pgfsetstrokecolor{eps2pgf_color}\pgfsetfillcolor{eps2pgf_color}
\pgfpathmoveto{\pgfqpoint{0.273cm}{1.451cm}}
\pgfpathcurveto{\pgfqpoint{0.273cm}{1.487cm}}{\pgfqpoint{0.259cm}{1.522cm}}{\pgfqpoint{0.233cm}{1.547cm}}
\pgfpathcurveto{\pgfqpoint{0.207cm}{1.573cm}}{\pgfqpoint{0.173cm}{1.588cm}}{\pgfqpoint{0.137cm}{1.588cm}}
\pgfpathcurveto{\pgfqpoint{0.1cm}{1.588cm}}{\pgfqpoint{0.066cm}{1.573cm}}{\pgfqpoint{0.04cm}{1.547cm}}
\pgfpathcurveto{\pgfqpoint{0.014cm}{1.522cm}}{\pgfqpoint{0cm}{1.487cm}}{\pgfqpoint{0cm}{1.451cm}}
\pgfpathcurveto{\pgfqpoint{0cm}{1.414cm}}{\pgfqpoint{0.014cm}{1.379cm}}{\pgfqpoint{0.04cm}{1.354cm}}
\pgfpathcurveto{\pgfqpoint{0.066cm}{1.328cm}}{\pgfqpoint{0.1cm}{1.314cm}}{\pgfqpoint{0.137cm}{1.314cm}}
\pgfpathcurveto{\pgfqpoint{0.173cm}{1.314cm}}{\pgfqpoint{0.207cm}{1.328cm}}{\pgfqpoint{0.233cm}{1.354cm}}
\pgfpathcurveto{\pgfqpoint{0.259cm}{1.379cm}}{\pgfqpoint{0.273cm}{1.414cm}}{\pgfqpoint{0.273cm}{1.451cm}}
\pgfusepath{fill}
\pgfpathmoveto{\pgfqpoint{1.345cm}{1.426cm}}
\pgfpathcurveto{\pgfqpoint{1.345cm}{1.463cm}}{\pgfqpoint{1.331cm}{1.497cm}}{\pgfqpoint{1.305cm}{1.523cm}}
\pgfpathcurveto{\pgfqpoint{1.28cm}{1.549cm}}{\pgfqpoint{1.245cm}{1.563cm}}{\pgfqpoint{1.209cm}{1.563cm}}
\pgfpathcurveto{\pgfqpoint{1.172cm}{1.563cm}}{\pgfqpoint{1.138cm}{1.549cm}}{\pgfqpoint{1.112cm}{1.523cm}}
\pgfpathcurveto{\pgfqpoint{1.087cm}{1.497cm}}{\pgfqpoint{1.072cm}{1.463cm}}{\pgfqpoint{1.072cm}{1.426cm}}
\pgfpathcurveto{\pgfqpoint{1.072cm}{1.39cm}}{\pgfqpoint{1.087cm}{1.355cm}}{\pgfqpoint{1.112cm}{1.329cm}}
\pgfpathcurveto{\pgfqpoint{1.138cm}{1.304cm}}{\pgfqpoint{1.172cm}{1.289cm}}{\pgfqpoint{1.209cm}{1.289cm}}
\pgfpathcurveto{\pgfqpoint{1.245cm}{1.289cm}}{\pgfqpoint{1.28cm}{1.304cm}}{\pgfqpoint{1.305cm}{1.329cm}}
\pgfpathcurveto{\pgfqpoint{1.331cm}{1.355cm}}{\pgfqpoint{1.345cm}{1.39cm}}{\pgfqpoint{1.345cm}{1.426cm}}
\pgfusepath{fill}
\begin{pgfscope}
\pgfsetdash{}{0cm}
\pgfsetlinewidth{0.818mm}
\pgfsetroundcap
\pgfsetmiterlimit{4.0}
\pgfpathmoveto{\pgfqpoint{0.682cm}{0.726cm}}
\pgfpathlineto{\pgfqpoint{0.682cm}{0.097cm}}
\pgfusepath{stroke}
\end{pgfscope}
\end{pgfscope}
\end{pgfscope}
\end{pgfscope}
\end{tikzpicture}}}$ in $\mathbb{X}_{M,
\varepsilon}$ since it can be controlled by $\llbracket X_{M, \varepsilon}^2
\rrbracket$ using Schauder estimates. 
In order to have a precise control of the number of copies of $X$ appearing in each  stochastic term we define $\|\mathbb X_{M,\varepsilon}\|$ as the smallest number bigger than 1 and all the quantities
\begin{equation}\label{eq:XX1}
 \| X_{M, \varepsilon} \|_{C_T \CC^{- 1 / 2 - \kappa, \varepsilon}
   (\rho^{\sigma})}, \quad \| \llbracket X_{M, \varepsilon}^2 \rrbracket
   \|^{1/2}_{C_T \CC^{- 1 - \kappa, \varepsilon} (\rho^{\sigma})}, \quad \|
   X_{M, \varepsilon}^{\!\resizebox{0.6em}{!}{
\begin{tikzpicture}
\pgfpathmoveto{\pgfqpoint{0cm}{-0.035cm}}
\pgfpathlineto{\pgfqpoint{1.376cm}{-0.035cm}}
\pgfpathlineto{\pgfqpoint{1.376cm}{1.552cm}}
\pgfpathlineto{\pgfqpoint{0cm}{1.552cm}}
\pgfpathclose
\pgfusepath{clip}
\begin{pgfscope}
\begin{pgfscope}
\pgfpathmoveto{\pgfqpoint{0cm}{-0.035cm}}
\pgfpathlineto{\pgfqpoint{1.376cm}{-0.035cm}}
\pgfpathlineto{\pgfqpoint{1.376cm}{1.552cm}}
\pgfpathlineto{\pgfqpoint{0cm}{1.552cm}}
\pgfpathclose
\pgfusepath{clip}
\begin{pgfscope}
\begin{pgfscope}
\pgfsetdash{}{0cm}
\pgfsetlinewidth{0.818mm}
\pgfsetroundcap
\pgfsetroundjoin
\pgfsetmiterlimit{7.0}
\definecolor{eps2pgf_color}{gray}{0}\pgfsetstrokecolor{eps2pgf_color}\pgfsetfillcolor{eps2pgf_color}
\pgfpathmoveto{\pgfqpoint{0.117cm}{1.421cm}}
\pgfpathlineto{\pgfqpoint{0.682cm}{0.671cm}}
\pgfpathlineto{\pgfqpoint{1.246cm}{1.421cm}}
\pgfusepath{stroke}
\end{pgfscope}
\definecolor{eps2pgf_color}{gray}{0}\pgfsetstrokecolor{eps2pgf_color}\pgfsetfillcolor{eps2pgf_color}
\pgfpathmoveto{\pgfqpoint{0.273cm}{1.395cm}}
\pgfpathcurveto{\pgfqpoint{0.273cm}{1.432cm}}{\pgfqpoint{0.259cm}{1.467cm}}{\pgfqpoint{0.233cm}{1.492cm}}
\pgfpathcurveto{\pgfqpoint{0.207cm}{1.518cm}}{\pgfqpoint{0.173cm}{1.532cm}}{\pgfqpoint{0.137cm}{1.532cm}}
\pgfpathcurveto{\pgfqpoint{0.1cm}{1.532cm}}{\pgfqpoint{0.066cm}{1.518cm}}{\pgfqpoint{0.04cm}{1.492cm}}
\pgfpathcurveto{\pgfqpoint{0.014cm}{1.467cm}}{\pgfqpoint{0cm}{1.432cm}}{\pgfqpoint{0cm}{1.395cm}}
\pgfpathcurveto{\pgfqpoint{0cm}{1.359cm}}{\pgfqpoint{0.014cm}{1.324cm}}{\pgfqpoint{0.04cm}{1.299cm}}
\pgfpathcurveto{\pgfqpoint{0.066cm}{1.273cm}}{\pgfqpoint{0.1cm}{1.258cm}}{\pgfqpoint{0.137cm}{1.258cm}}
\pgfpathcurveto{\pgfqpoint{0.173cm}{1.258cm}}{\pgfqpoint{0.207cm}{1.273cm}}{\pgfqpoint{0.233cm}{1.299cm}}
\pgfpathcurveto{\pgfqpoint{0.259cm}{1.324cm}}{\pgfqpoint{0.273cm}{1.359cm}}{\pgfqpoint{0.273cm}{1.395cm}}
\pgfusepath{fill}
\begin{pgfscope}
\pgfsetdash{}{0cm}
\pgfsetlinewidth{0.818mm}
\pgfsetmiterlimit{7.0}
\pgfpathmoveto{\pgfqpoint{0.682cm}{0.671cm}}
\pgfpathlineto{\pgfqpoint{0.679cm}{1.418cm}}
\pgfusepath{stroke}
\end{pgfscope}
\pgfpathmoveto{\pgfqpoint{0.815cm}{1.399cm}}
\pgfpathcurveto{\pgfqpoint{0.815cm}{1.435cm}}{\pgfqpoint{0.801cm}{1.47cm}}{\pgfqpoint{0.775cm}{1.496cm}}
\pgfpathcurveto{\pgfqpoint{0.75cm}{1.521cm}}{\pgfqpoint{0.715cm}{1.536cm}}{\pgfqpoint{0.679cm}{1.536cm}}
\pgfpathcurveto{\pgfqpoint{0.643cm}{1.536cm}}{\pgfqpoint{0.608cm}{1.521cm}}{\pgfqpoint{0.582cm}{1.496cm}}
\pgfpathcurveto{\pgfqpoint{0.557cm}{1.47cm}}{\pgfqpoint{0.542cm}{1.435cm}}{\pgfqpoint{0.542cm}{1.399cm}}
\pgfpathcurveto{\pgfqpoint{0.542cm}{1.363cm}}{\pgfqpoint{0.557cm}{1.328cm}}{\pgfqpoint{0.582cm}{1.302cm}}
\pgfpathcurveto{\pgfqpoint{0.608cm}{1.276cm}}{\pgfqpoint{0.643cm}{1.262cm}}{\pgfqpoint{0.679cm}{1.262cm}}
\pgfpathcurveto{\pgfqpoint{0.715cm}{1.262cm}}{\pgfqpoint{0.75cm}{1.276cm}}{\pgfqpoint{0.775cm}{1.302cm}}
\pgfpathcurveto{\pgfqpoint{0.801cm}{1.328cm}}{\pgfqpoint{0.815cm}{1.363cm}}{\pgfqpoint{0.815cm}{1.399cm}}
\pgfusepath{fill}
\pgfpathmoveto{\pgfqpoint{1.345cm}{1.371cm}}
\pgfpathcurveto{\pgfqpoint{1.345cm}{1.408cm}}{\pgfqpoint{1.331cm}{1.442cm}}{\pgfqpoint{1.305cm}{1.468cm}}
\pgfpathcurveto{\pgfqpoint{1.28cm}{1.494cm}}{\pgfqpoint{1.245cm}{1.508cm}}{\pgfqpoint{1.209cm}{1.508cm}}
\pgfpathcurveto{\pgfqpoint{1.172cm}{1.508cm}}{\pgfqpoint{1.138cm}{1.494cm}}{\pgfqpoint{1.112cm}{1.468cm}}
\pgfpathcurveto{\pgfqpoint{1.087cm}{1.442cm}}{\pgfqpoint{1.072cm}{1.408cm}}{\pgfqpoint{1.072cm}{1.371cm}}
\pgfpathcurveto{\pgfqpoint{1.072cm}{1.335cm}}{\pgfqpoint{1.087cm}{1.3cm}}{\pgfqpoint{1.112cm}{1.274cm}}
\pgfpathcurveto{\pgfqpoint{1.138cm}{1.249cm}}{\pgfqpoint{1.172cm}{1.234cm}}{\pgfqpoint{1.209cm}{1.234cm}}
\pgfpathcurveto{\pgfqpoint{1.245cm}{1.234cm}}{\pgfqpoint{1.28cm}{1.249cm}}{\pgfqpoint{1.305cm}{1.274cm}}
\pgfpathcurveto{\pgfqpoint{1.331cm}{1.3cm}}{\pgfqpoint{1.345cm}{1.335cm}}{\pgfqpoint{1.345cm}{1.371cm}}
\pgfusepath{fill}
\begin{pgfscope}
\pgfsetdash{}{0cm}
\pgfsetlinewidth{0.818mm}
\pgfsetroundcap
\pgfsetmiterlimit{4.0}
\pgfpathmoveto{\pgfqpoint{0.682cm}{0.671cm}}
\pgfpathlineto{\pgfqpoint{0.682cm}{0.042cm}}
\pgfusepath{stroke}
\end{pgfscope}
\end{pgfscope}
\end{pgfscope}
\end{pgfscope}
\end{tikzpicture}}} \|^{1/3}_{C_T \CC^{1 / 2 - \kappa,
   \varepsilon} (\rho^{\sigma})},
   \end{equation}
\begin{equation}\label{eq:XX2}
 \quad \| X_{M, \varepsilon}^{\!\resizebox{0.6em}{!}{
\begin{tikzpicture}
\pgfpathmoveto{\pgfqpoint{0cm}{-0.035cm}}
\pgfpathlineto{\pgfqpoint{1.376cm}{-0.035cm}}
\pgfpathlineto{\pgfqpoint{1.376cm}{1.552cm}}
\pgfpathlineto{\pgfqpoint{0cm}{1.552cm}}
\pgfpathclose
\pgfusepath{clip}
\begin{pgfscope}
\begin{pgfscope}
\pgfpathmoveto{\pgfqpoint{0cm}{-0.035cm}}
\pgfpathlineto{\pgfqpoint{1.376cm}{-0.035cm}}
\pgfpathlineto{\pgfqpoint{1.376cm}{1.552cm}}
\pgfpathlineto{\pgfqpoint{0cm}{1.552cm}}
\pgfpathclose
\pgfusepath{clip}
\begin{pgfscope}
\begin{pgfscope}
\pgfsetdash{}{0cm}
\pgfsetlinewidth{0.818mm}
\pgfsetroundcap
\pgfsetroundjoin
\pgfsetmiterlimit{7.0}
\definecolor{eps2pgf_color}{gray}{0}\pgfsetstrokecolor{eps2pgf_color}\pgfsetfillcolor{eps2pgf_color}
\pgfpathmoveto{\pgfqpoint{0.117cm}{1.421cm}}
\pgfpathlineto{\pgfqpoint{0.682cm}{0.671cm}}
\pgfpathlineto{\pgfqpoint{1.246cm}{1.421cm}}
\pgfusepath{stroke}
\end{pgfscope}
\definecolor{eps2pgf_color}{gray}{0}\pgfsetstrokecolor{eps2pgf_color}\pgfsetfillcolor{eps2pgf_color}
\pgfpathmoveto{\pgfqpoint{0.273cm}{1.395cm}}
\pgfpathcurveto{\pgfqpoint{0.273cm}{1.432cm}}{\pgfqpoint{0.259cm}{1.467cm}}{\pgfqpoint{0.233cm}{1.492cm}}
\pgfpathcurveto{\pgfqpoint{0.207cm}{1.518cm}}{\pgfqpoint{0.173cm}{1.532cm}}{\pgfqpoint{0.137cm}{1.532cm}}
\pgfpathcurveto{\pgfqpoint{0.1cm}{1.532cm}}{\pgfqpoint{0.066cm}{1.518cm}}{\pgfqpoint{0.04cm}{1.492cm}}
\pgfpathcurveto{\pgfqpoint{0.014cm}{1.467cm}}{\pgfqpoint{0cm}{1.432cm}}{\pgfqpoint{0cm}{1.395cm}}
\pgfpathcurveto{\pgfqpoint{0cm}{1.359cm}}{\pgfqpoint{0.014cm}{1.324cm}}{\pgfqpoint{0.04cm}{1.299cm}}
\pgfpathcurveto{\pgfqpoint{0.066cm}{1.273cm}}{\pgfqpoint{0.1cm}{1.258cm}}{\pgfqpoint{0.137cm}{1.258cm}}
\pgfpathcurveto{\pgfqpoint{0.173cm}{1.258cm}}{\pgfqpoint{0.207cm}{1.273cm}}{\pgfqpoint{0.233cm}{1.299cm}}
\pgfpathcurveto{\pgfqpoint{0.259cm}{1.324cm}}{\pgfqpoint{0.273cm}{1.359cm}}{\pgfqpoint{0.273cm}{1.395cm}}
\pgfusepath{fill}
\begin{pgfscope}
\pgfsetdash{}{0cm}
\pgfsetlinewidth{0.818mm}
\pgfsetmiterlimit{7.0}
\pgfpathmoveto{\pgfqpoint{0.682cm}{0.671cm}}
\pgfpathlineto{\pgfqpoint{0.679cm}{1.418cm}}
\pgfusepath{stroke}
\end{pgfscope}
\pgfpathmoveto{\pgfqpoint{0.815cm}{1.399cm}}
\pgfpathcurveto{\pgfqpoint{0.815cm}{1.435cm}}{\pgfqpoint{0.801cm}{1.47cm}}{\pgfqpoint{0.775cm}{1.496cm}}
\pgfpathcurveto{\pgfqpoint{0.75cm}{1.521cm}}{\pgfqpoint{0.715cm}{1.536cm}}{\pgfqpoint{0.679cm}{1.536cm}}
\pgfpathcurveto{\pgfqpoint{0.643cm}{1.536cm}}{\pgfqpoint{0.608cm}{1.521cm}}{\pgfqpoint{0.582cm}{1.496cm}}
\pgfpathcurveto{\pgfqpoint{0.557cm}{1.47cm}}{\pgfqpoint{0.542cm}{1.435cm}}{\pgfqpoint{0.542cm}{1.399cm}}
\pgfpathcurveto{\pgfqpoint{0.542cm}{1.363cm}}{\pgfqpoint{0.557cm}{1.328cm}}{\pgfqpoint{0.582cm}{1.302cm}}
\pgfpathcurveto{\pgfqpoint{0.608cm}{1.276cm}}{\pgfqpoint{0.643cm}{1.262cm}}{\pgfqpoint{0.679cm}{1.262cm}}
\pgfpathcurveto{\pgfqpoint{0.715cm}{1.262cm}}{\pgfqpoint{0.75cm}{1.276cm}}{\pgfqpoint{0.775cm}{1.302cm}}
\pgfpathcurveto{\pgfqpoint{0.801cm}{1.328cm}}{\pgfqpoint{0.815cm}{1.363cm}}{\pgfqpoint{0.815cm}{1.399cm}}
\pgfusepath{fill}
\pgfpathmoveto{\pgfqpoint{1.345cm}{1.371cm}}
\pgfpathcurveto{\pgfqpoint{1.345cm}{1.408cm}}{\pgfqpoint{1.331cm}{1.442cm}}{\pgfqpoint{1.305cm}{1.468cm}}
\pgfpathcurveto{\pgfqpoint{1.28cm}{1.494cm}}{\pgfqpoint{1.245cm}{1.508cm}}{\pgfqpoint{1.209cm}{1.508cm}}
\pgfpathcurveto{\pgfqpoint{1.172cm}{1.508cm}}{\pgfqpoint{1.138cm}{1.494cm}}{\pgfqpoint{1.112cm}{1.468cm}}
\pgfpathcurveto{\pgfqpoint{1.087cm}{1.442cm}}{\pgfqpoint{1.072cm}{1.408cm}}{\pgfqpoint{1.072cm}{1.371cm}}
\pgfpathcurveto{\pgfqpoint{1.072cm}{1.335cm}}{\pgfqpoint{1.087cm}{1.3cm}}{\pgfqpoint{1.112cm}{1.274cm}}
\pgfpathcurveto{\pgfqpoint{1.138cm}{1.249cm}}{\pgfqpoint{1.172cm}{1.234cm}}{\pgfqpoint{1.209cm}{1.234cm}}
\pgfpathcurveto{\pgfqpoint{1.245cm}{1.234cm}}{\pgfqpoint{1.28cm}{1.249cm}}{\pgfqpoint{1.305cm}{1.274cm}}
\pgfpathcurveto{\pgfqpoint{1.331cm}{1.3cm}}{\pgfqpoint{1.345cm}{1.335cm}}{\pgfqpoint{1.345cm}{1.371cm}}
\pgfusepath{fill}
\begin{pgfscope}
\pgfsetdash{}{0cm}
\pgfsetlinewidth{0.818mm}
\pgfsetroundcap
\pgfsetmiterlimit{4.0}
\pgfpathmoveto{\pgfqpoint{0.682cm}{0.671cm}}
\pgfpathlineto{\pgfqpoint{0.682cm}{0.042cm}}
\pgfusepath{stroke}
\end{pgfscope}
\end{pgfscope}
\end{pgfscope}
\end{pgfscope}
\end{tikzpicture}}} \|^{1/3}_{C_{T}^{\beta / 2}
   L^{\infty, \varepsilon} (\rho^{\sigma})}, \qquad \| X_{M,
   \varepsilon}^{\!\resizebox{!}{.8em}{
\begin{tikzpicture}
\pgfpathmoveto{\pgfqpoint{0cm}{-0.035cm}}
\pgfpathlineto{\pgfqpoint{1.976cm}{-0.035cm}}
\pgfpathlineto{\pgfqpoint{1.976cm}{1.94cm}}
\pgfpathlineto{\pgfqpoint{0cm}{1.94cm}}
\pgfpathclose
\pgfusepath{clip}
\begin{pgfscope}
\begin{pgfscope}
\pgfpathmoveto{\pgfqpoint{0cm}{-0.035cm}}
\pgfpathlineto{\pgfqpoint{1.976cm}{-0.035cm}}
\pgfpathlineto{\pgfqpoint{1.976cm}{1.94cm}}
\pgfpathlineto{\pgfqpoint{0cm}{1.94cm}}
\pgfpathclose
\pgfusepath{clip}
\begin{pgfscope}
\begin{pgfscope}
\pgfsetdash{}{0cm}
\pgfsetlinewidth{0.818mm}
\pgfsetroundcap
\pgfsetroundjoin
\pgfsetmiterlimit{7.0}
\definecolor{eps2pgf_color}{gray}{0}\pgfsetstrokecolor{eps2pgf_color}\pgfsetfillcolor{eps2pgf_color}
\pgfpathmoveto{\pgfqpoint{0.117cm}{1.815cm}}
\pgfpathlineto{\pgfqpoint{0.682cm}{1.065cm}}
\pgfpathlineto{\pgfqpoint{1.246cm}{1.815cm}}
\pgfusepath{stroke}
\end{pgfscope}
\definecolor{eps2pgf_color}{gray}{0}\pgfsetstrokecolor{eps2pgf_color}\pgfsetfillcolor{eps2pgf_color}
\pgfpathmoveto{\pgfqpoint{0.273cm}{1.789cm}}
\pgfpathcurveto{\pgfqpoint{0.273cm}{1.825cm}}{\pgfqpoint{0.259cm}{1.86cm}}{\pgfqpoint{0.233cm}{1.886cm}}
\pgfpathcurveto{\pgfqpoint{0.207cm}{1.912cm}}{\pgfqpoint{0.173cm}{1.926cm}}{\pgfqpoint{0.137cm}{1.926cm}}
\pgfpathcurveto{\pgfqpoint{0.1cm}{1.926cm}}{\pgfqpoint{0.066cm}{1.912cm}}{\pgfqpoint{0.04cm}{1.886cm}}
\pgfpathcurveto{\pgfqpoint{0.014cm}{1.86cm}}{\pgfqpoint{0cm}{1.825cm}}{\pgfqpoint{0cm}{1.789cm}}
\pgfpathcurveto{\pgfqpoint{0cm}{1.753cm}}{\pgfqpoint{0.014cm}{1.718cm}}{\pgfqpoint{0.04cm}{1.692cm}}
\pgfpathcurveto{\pgfqpoint{0.066cm}{1.667cm}}{\pgfqpoint{0.1cm}{1.652cm}}{\pgfqpoint{0.137cm}{1.652cm}}
\pgfpathcurveto{\pgfqpoint{0.173cm}{1.652cm}}{\pgfqpoint{0.207cm}{1.667cm}}{\pgfqpoint{0.233cm}{1.692cm}}
\pgfpathcurveto{\pgfqpoint{0.259cm}{1.718cm}}{\pgfqpoint{0.273cm}{1.753cm}}{\pgfqpoint{0.273cm}{1.789cm}}
\pgfusepath{fill}
\begin{pgfscope}
\pgfsetdash{}{0cm}
\pgfsetlinewidth{0.818mm}
\pgfsetmiterlimit{7.0}
\pgfpathmoveto{\pgfqpoint{0.682cm}{1.065cm}}
\pgfpathlineto{\pgfqpoint{0.679cm}{1.812cm}}
\pgfusepath{stroke}
\end{pgfscope}
\pgfpathmoveto{\pgfqpoint{0.815cm}{1.793cm}}
\pgfpathcurveto{\pgfqpoint{0.815cm}{1.829cm}}{\pgfqpoint{0.801cm}{1.864cm}}{\pgfqpoint{0.775cm}{1.89cm}}
\pgfpathcurveto{\pgfqpoint{0.75cm}{1.915cm}}{\pgfqpoint{0.715cm}{1.93cm}}{\pgfqpoint{0.679cm}{1.93cm}}
\pgfpathcurveto{\pgfqpoint{0.643cm}{1.93cm}}{\pgfqpoint{0.608cm}{1.915cm}}{\pgfqpoint{0.582cm}{1.89cm}}
\pgfpathcurveto{\pgfqpoint{0.557cm}{1.864cm}}{\pgfqpoint{0.542cm}{1.829cm}}{\pgfqpoint{0.542cm}{1.793cm}}
\pgfpathcurveto{\pgfqpoint{0.542cm}{1.756cm}}{\pgfqpoint{0.557cm}{1.722cm}}{\pgfqpoint{0.582cm}{1.696cm}}
\pgfpathcurveto{\pgfqpoint{0.608cm}{1.67cm}}{\pgfqpoint{0.643cm}{1.656cm}}{\pgfqpoint{0.679cm}{1.656cm}}
\pgfpathcurveto{\pgfqpoint{0.715cm}{1.656cm}}{\pgfqpoint{0.75cm}{1.67cm}}{\pgfqpoint{0.775cm}{1.696cm}}
\pgfpathcurveto{\pgfqpoint{0.801cm}{1.722cm}}{\pgfqpoint{0.815cm}{1.756cm}}{\pgfqpoint{0.815cm}{1.793cm}}
\pgfusepath{fill}
\pgfpathmoveto{\pgfqpoint{1.345cm}{1.765cm}}
\pgfpathcurveto{\pgfqpoint{1.345cm}{1.801cm}}{\pgfqpoint{1.331cm}{1.836cm}}{\pgfqpoint{1.305cm}{1.862cm}}
\pgfpathcurveto{\pgfqpoint{1.28cm}{1.887cm}}{\pgfqpoint{1.245cm}{1.902cm}}{\pgfqpoint{1.209cm}{1.902cm}}
\pgfpathcurveto{\pgfqpoint{1.172cm}{1.902cm}}{\pgfqpoint{1.138cm}{1.887cm}}{\pgfqpoint{1.112cm}{1.862cm}}
\pgfpathcurveto{\pgfqpoint{1.087cm}{1.836cm}}{\pgfqpoint{1.072cm}{1.801cm}}{\pgfqpoint{1.072cm}{1.765cm}}
\pgfpathcurveto{\pgfqpoint{1.072cm}{1.728cm}}{\pgfqpoint{1.087cm}{1.694cm}}{\pgfqpoint{1.112cm}{1.668cm}}
\pgfpathcurveto{\pgfqpoint{1.138cm}{1.642cm}}{\pgfqpoint{1.172cm}{1.628cm}}{\pgfqpoint{1.209cm}{1.628cm}}
\pgfpathcurveto{\pgfqpoint{1.245cm}{1.628cm}}{\pgfqpoint{1.28cm}{1.642cm}}{\pgfqpoint{1.305cm}{1.668cm}}
\pgfpathcurveto{\pgfqpoint{1.331cm}{1.694cm}}{\pgfqpoint{1.345cm}{1.728cm}}{\pgfqpoint{1.345cm}{1.765cm}}
\pgfusepath{fill}
\begin{pgfscope}
\pgfsetdash{}{0cm}
\pgfsetlinewidth{0.818mm}
\pgfsetroundcap
\pgfsetroundjoin
\pgfsetmiterlimit{7.0}
\pgfpathmoveto{\pgfqpoint{0.682cm}{1.065cm}}
\pgfpathlineto{\pgfqpoint{1.246cm}{0.315cm}}
\pgfpathlineto{\pgfqpoint{1.811cm}{1.065cm}}
\pgfusepath{stroke}
\end{pgfscope}
\pgfpathmoveto{\pgfqpoint{1.948cm}{1.065cm}}
\pgfpathcurveto{\pgfqpoint{1.948cm}{1.101cm}}{\pgfqpoint{1.933cm}{1.136cm}}{\pgfqpoint{1.907cm}{1.162cm}}
\pgfpathcurveto{\pgfqpoint{1.882cm}{1.187cm}}{\pgfqpoint{1.847cm}{1.202cm}}{\pgfqpoint{1.811cm}{1.202cm}}
\pgfpathcurveto{\pgfqpoint{1.775cm}{1.202cm}}{\pgfqpoint{1.74cm}{1.187cm}}{\pgfqpoint{1.714cm}{1.162cm}}
\pgfpathcurveto{\pgfqpoint{1.689cm}{1.136cm}}{\pgfqpoint{1.674cm}{1.101cm}}{\pgfqpoint{1.674cm}{1.065cm}}
\pgfpathcurveto{\pgfqpoint{1.674cm}{1.029cm}}{\pgfqpoint{1.689cm}{0.994cm}}{\pgfqpoint{1.714cm}{0.968cm}}
\pgfpathcurveto{\pgfqpoint{1.74cm}{0.942cm}}{\pgfqpoint{1.775cm}{0.928cm}}{\pgfqpoint{1.811cm}{0.928cm}}
\pgfpathcurveto{\pgfqpoint{1.847cm}{0.928cm}}{\pgfqpoint{1.882cm}{0.942cm}}{\pgfqpoint{1.907cm}{0.968cm}}
\pgfpathcurveto{\pgfqpoint{1.933cm}{0.994cm}}{\pgfqpoint{1.948cm}{1.029cm}}{\pgfqpoint{1.948cm}{1.065cm}}
\pgfusepath{fill}
\begin{pgfscope}
\pgfsetdash{}{0cm}
\pgfsetlinewidth{0.818mm}
\pgfsetmiterlimit{4.0}
\pgfpathmoveto{\pgfqpoint{1.383cm}{0.178cm}}
\pgfpathcurveto{\pgfqpoint{1.383cm}{0.214cm}}{\pgfqpoint{1.369cm}{0.249cm}}{\pgfqpoint{1.343cm}{0.275cm}}
\pgfpathcurveto{\pgfqpoint{1.317cm}{0.3cm}}{\pgfqpoint{1.283cm}{0.315cm}}{\pgfqpoint{1.246cm}{0.315cm}}
\pgfpathcurveto{\pgfqpoint{1.21cm}{0.315cm}}{\pgfqpoint{1.175cm}{0.3cm}}{\pgfqpoint{1.15cm}{0.275cm}}
\pgfpathcurveto{\pgfqpoint{1.124cm}{0.249cm}}{\pgfqpoint{1.11cm}{0.214cm}}{\pgfqpoint{1.11cm}{0.178cm}}
\pgfpathcurveto{\pgfqpoint{1.11cm}{0.141cm}}{\pgfqpoint{1.124cm}{0.107cm}}{\pgfqpoint{1.15cm}{0.081cm}}
\pgfpathcurveto{\pgfqpoint{1.175cm}{0.055cm}}{\pgfqpoint{1.21cm}{0.041cm}}{\pgfqpoint{1.246cm}{0.041cm}}
\pgfpathcurveto{\pgfqpoint{1.283cm}{0.041cm}}{\pgfqpoint{1.317cm}{0.055cm}}{\pgfqpoint{1.343cm}{0.081cm}}
\pgfpathcurveto{\pgfqpoint{1.369cm}{0.107cm}}{\pgfqpoint{1.383cm}{0.141cm}}{\pgfqpoint{1.383cm}{0.178cm}}
\pgfusepath{stroke}
\end{pgfscope}
\end{pgfscope}
\end{pgfscope}
\end{pgfscope}
\end{tikzpicture}}} \|^{1/4}_{C_T \CC^{- \kappa, \varepsilon} (\rho^{\sigma})},
\end{equation}
\begin{equation}\label{eq:XX3}
 \| X_{M, \varepsilon}^{\!\resizebox{!}{.8em}{
\begin{tikzpicture}
\pgfpathmoveto{\pgfqpoint{0cm}{-0.035cm}}
\pgfpathlineto{\pgfqpoint{1.976cm}{-0.035cm}}
\pgfpathlineto{\pgfqpoint{1.976cm}{1.94cm}}
\pgfpathlineto{\pgfqpoint{0cm}{1.94cm}}
\pgfpathclose
\pgfusepath{clip}
\begin{pgfscope}
\begin{pgfscope}
\pgfpathmoveto{\pgfqpoint{0cm}{-0.035cm}}
\pgfpathlineto{\pgfqpoint{1.976cm}{-0.035cm}}
\pgfpathlineto{\pgfqpoint{1.976cm}{1.94cm}}
\pgfpathlineto{\pgfqpoint{0cm}{1.94cm}}
\pgfpathclose
\pgfusepath{clip}
\begin{pgfscope}
\begin{pgfscope}
\pgfsetdash{}{0cm}
\pgfsetlinewidth{0.818mm}
\pgfsetroundcap
\pgfsetroundjoin
\pgfsetmiterlimit{7.0}
\definecolor{eps2pgf_color}{gray}{0}\pgfsetstrokecolor{eps2pgf_color}\pgfsetfillcolor{eps2pgf_color}
\pgfpathmoveto{\pgfqpoint{0.117cm}{1.815cm}}
\pgfpathlineto{\pgfqpoint{0.682cm}{1.065cm}}
\pgfpathlineto{\pgfqpoint{1.246cm}{1.815cm}}
\pgfusepath{stroke}
\end{pgfscope}
\definecolor{eps2pgf_color}{gray}{0}\pgfsetstrokecolor{eps2pgf_color}\pgfsetfillcolor{eps2pgf_color}
\pgfpathmoveto{\pgfqpoint{0.273cm}{1.789cm}}
\pgfpathcurveto{\pgfqpoint{0.273cm}{1.825cm}}{\pgfqpoint{0.259cm}{1.86cm}}{\pgfqpoint{0.233cm}{1.886cm}}
\pgfpathcurveto{\pgfqpoint{0.207cm}{1.912cm}}{\pgfqpoint{0.173cm}{1.926cm}}{\pgfqpoint{0.137cm}{1.926cm}}
\pgfpathcurveto{\pgfqpoint{0.1cm}{1.926cm}}{\pgfqpoint{0.066cm}{1.912cm}}{\pgfqpoint{0.04cm}{1.886cm}}
\pgfpathcurveto{\pgfqpoint{0.014cm}{1.86cm}}{\pgfqpoint{0cm}{1.825cm}}{\pgfqpoint{0cm}{1.789cm}}
\pgfpathcurveto{\pgfqpoint{0cm}{1.753cm}}{\pgfqpoint{0.014cm}{1.718cm}}{\pgfqpoint{0.04cm}{1.692cm}}
\pgfpathcurveto{\pgfqpoint{0.066cm}{1.667cm}}{\pgfqpoint{0.1cm}{1.652cm}}{\pgfqpoint{0.137cm}{1.652cm}}
\pgfpathcurveto{\pgfqpoint{0.173cm}{1.652cm}}{\pgfqpoint{0.207cm}{1.667cm}}{\pgfqpoint{0.233cm}{1.692cm}}
\pgfpathcurveto{\pgfqpoint{0.259cm}{1.718cm}}{\pgfqpoint{0.273cm}{1.753cm}}{\pgfqpoint{0.273cm}{1.789cm}}
\pgfusepath{fill}
\pgfpathmoveto{\pgfqpoint{1.345cm}{1.765cm}}
\pgfpathcurveto{\pgfqpoint{1.345cm}{1.801cm}}{\pgfqpoint{1.331cm}{1.836cm}}{\pgfqpoint{1.305cm}{1.862cm}}
\pgfpathcurveto{\pgfqpoint{1.28cm}{1.887cm}}{\pgfqpoint{1.245cm}{1.902cm}}{\pgfqpoint{1.209cm}{1.902cm}}
\pgfpathcurveto{\pgfqpoint{1.172cm}{1.902cm}}{\pgfqpoint{1.138cm}{1.887cm}}{\pgfqpoint{1.112cm}{1.862cm}}
\pgfpathcurveto{\pgfqpoint{1.087cm}{1.836cm}}{\pgfqpoint{1.072cm}{1.801cm}}{\pgfqpoint{1.072cm}{1.765cm}}
\pgfpathcurveto{\pgfqpoint{1.072cm}{1.728cm}}{\pgfqpoint{1.087cm}{1.694cm}}{\pgfqpoint{1.112cm}{1.668cm}}
\pgfpathcurveto{\pgfqpoint{1.138cm}{1.642cm}}{\pgfqpoint{1.172cm}{1.628cm}}{\pgfqpoint{1.209cm}{1.628cm}}
\pgfpathcurveto{\pgfqpoint{1.245cm}{1.628cm}}{\pgfqpoint{1.28cm}{1.642cm}}{\pgfqpoint{1.305cm}{1.668cm}}
\pgfpathcurveto{\pgfqpoint{1.331cm}{1.694cm}}{\pgfqpoint{1.345cm}{1.728cm}}{\pgfqpoint{1.345cm}{1.765cm}}
\pgfusepath{fill}
\begin{pgfscope}
\pgfsetdash{}{0cm}
\pgfsetlinewidth{0.818mm}
\pgfsetroundcap
\pgfsetroundjoin
\pgfsetmiterlimit{7.0}
\pgfpathmoveto{\pgfqpoint{0.682cm}{1.065cm}}
\pgfpathlineto{\pgfqpoint{1.246cm}{0.315cm}}
\pgfpathlineto{\pgfqpoint{1.811cm}{1.065cm}}
\pgfusepath{stroke}
\end{pgfscope}
\pgfpathmoveto{\pgfqpoint{1.948cm}{1.065cm}}
\pgfpathcurveto{\pgfqpoint{1.948cm}{1.101cm}}{\pgfqpoint{1.933cm}{1.136cm}}{\pgfqpoint{1.907cm}{1.162cm}}
\pgfpathcurveto{\pgfqpoint{1.882cm}{1.187cm}}{\pgfqpoint{1.847cm}{1.202cm}}{\pgfqpoint{1.811cm}{1.202cm}}
\pgfpathcurveto{\pgfqpoint{1.775cm}{1.202cm}}{\pgfqpoint{1.74cm}{1.187cm}}{\pgfqpoint{1.714cm}{1.162cm}}
\pgfpathcurveto{\pgfqpoint{1.689cm}{1.136cm}}{\pgfqpoint{1.674cm}{1.101cm}}{\pgfqpoint{1.674cm}{1.065cm}}
\pgfpathcurveto{\pgfqpoint{1.674cm}{1.029cm}}{\pgfqpoint{1.689cm}{0.994cm}}{\pgfqpoint{1.714cm}{0.968cm}}
\pgfpathcurveto{\pgfqpoint{1.74cm}{0.942cm}}{\pgfqpoint{1.775cm}{0.928cm}}{\pgfqpoint{1.811cm}{0.928cm}}
\pgfpathcurveto{\pgfqpoint{1.847cm}{0.928cm}}{\pgfqpoint{1.882cm}{0.942cm}}{\pgfqpoint{1.907cm}{0.968cm}}
\pgfpathcurveto{\pgfqpoint{1.933cm}{0.994cm}}{\pgfqpoint{1.948cm}{1.029cm}}{\pgfqpoint{1.948cm}{1.065cm}}
\pgfusepath{fill}
\begin{pgfscope}
\pgfsetdash{}{0cm}
\pgfsetlinewidth{0.818mm}
\pgfsetmiterlimit{7.0}
\pgfpathmoveto{\pgfqpoint{1.246cm}{0.315cm}}
\pgfpathlineto{\pgfqpoint{1.244cm}{1.061cm}}
\pgfusepath{stroke}
\end{pgfscope}
\pgfpathmoveto{\pgfqpoint{1.38cm}{1.065cm}}
\pgfpathcurveto{\pgfqpoint{1.38cm}{1.101cm}}{\pgfqpoint{1.366cm}{1.136cm}}{\pgfqpoint{1.34cm}{1.162cm}}
\pgfpathcurveto{\pgfqpoint{1.315cm}{1.187cm}}{\pgfqpoint{1.28cm}{1.202cm}}{\pgfqpoint{1.244cm}{1.202cm}}
\pgfpathcurveto{\pgfqpoint{1.207cm}{1.202cm}}{\pgfqpoint{1.173cm}{1.187cm}}{\pgfqpoint{1.147cm}{1.162cm}}
\pgfpathcurveto{\pgfqpoint{1.121cm}{1.136cm}}{\pgfqpoint{1.107cm}{1.101cm}}{\pgfqpoint{1.107cm}{1.065cm}}
\pgfpathcurveto{\pgfqpoint{1.107cm}{1.029cm}}{\pgfqpoint{1.121cm}{0.994cm}}{\pgfqpoint{1.147cm}{0.968cm}}
\pgfpathcurveto{\pgfqpoint{1.173cm}{0.942cm}}{\pgfqpoint{1.207cm}{0.928cm}}{\pgfqpoint{1.244cm}{0.928cm}}
\pgfpathcurveto{\pgfqpoint{1.28cm}{0.928cm}}{\pgfqpoint{1.315cm}{0.942cm}}{\pgfqpoint{1.34cm}{0.968cm}}
\pgfpathcurveto{\pgfqpoint{1.366cm}{0.994cm}}{\pgfqpoint{1.38cm}{1.029cm}}{\pgfqpoint{1.38cm}{1.065cm}}
\pgfusepath{fill}
\begin{pgfscope}
\pgfsetdash{}{0cm}
\pgfsetlinewidth{0.818mm}
\pgfsetmiterlimit{4.0}
\pgfpathmoveto{\pgfqpoint{1.383cm}{0.178cm}}
\pgfpathcurveto{\pgfqpoint{1.383cm}{0.214cm}}{\pgfqpoint{1.369cm}{0.249cm}}{\pgfqpoint{1.343cm}{0.275cm}}
\pgfpathcurveto{\pgfqpoint{1.317cm}{0.3cm}}{\pgfqpoint{1.283cm}{0.315cm}}{\pgfqpoint{1.246cm}{0.315cm}}
\pgfpathcurveto{\pgfqpoint{1.21cm}{0.315cm}}{\pgfqpoint{1.175cm}{0.3cm}}{\pgfqpoint{1.15cm}{0.275cm}}
\pgfpathcurveto{\pgfqpoint{1.124cm}{0.249cm}}{\pgfqpoint{1.11cm}{0.214cm}}{\pgfqpoint{1.11cm}{0.178cm}}
\pgfpathcurveto{\pgfqpoint{1.11cm}{0.141cm}}{\pgfqpoint{1.124cm}{0.107cm}}{\pgfqpoint{1.15cm}{0.081cm}}
\pgfpathcurveto{\pgfqpoint{1.175cm}{0.055cm}}{\pgfqpoint{1.21cm}{0.041cm}}{\pgfqpoint{1.246cm}{0.041cm}}
\pgfpathcurveto{\pgfqpoint{1.283cm}{0.041cm}}{\pgfqpoint{1.317cm}{0.055cm}}{\pgfqpoint{1.343cm}{0.081cm}}
\pgfpathcurveto{\pgfqpoint{1.369cm}{0.107cm}}{\pgfqpoint{1.383cm}{0.141cm}}{\pgfqpoint{1.383cm}{0.178cm}}
\pgfusepath{stroke}
\end{pgfscope}
\end{pgfscope}
\end{pgfscope}
\end{pgfscope}
\end{tikzpicture}}} \|^{1/4}_{C_{T} \CC^{- \kappa, \varepsilon}
   (\rho^{\sigma})}, \quad \| \tilde{X}_{M, \varepsilon}^{\!\resizebox{!}{.8em}{
\begin{tikzpicture}
\pgfpathmoveto{\pgfqpoint{0cm}{-0.035cm}}
\pgfpathlineto{\pgfqpoint{1.976cm}{-0.035cm}}
\pgfpathlineto{\pgfqpoint{1.976cm}{1.94cm}}
\pgfpathlineto{\pgfqpoint{0cm}{1.94cm}}
\pgfpathclose
\pgfusepath{clip}
\begin{pgfscope}
\begin{pgfscope}
\pgfpathmoveto{\pgfqpoint{0cm}{-0.035cm}}
\pgfpathlineto{\pgfqpoint{1.976cm}{-0.035cm}}
\pgfpathlineto{\pgfqpoint{1.976cm}{1.94cm}}
\pgfpathlineto{\pgfqpoint{0cm}{1.94cm}}
\pgfpathclose
\pgfusepath{clip}
\begin{pgfscope}
\begin{pgfscope}
\pgfsetdash{}{0cm}
\pgfsetlinewidth{0.818mm}
\pgfsetroundcap
\pgfsetroundjoin
\pgfsetmiterlimit{7.0}
\definecolor{eps2pgf_color}{gray}{0}\pgfsetstrokecolor{eps2pgf_color}\pgfsetfillcolor{eps2pgf_color}
\pgfpathmoveto{\pgfqpoint{0.117cm}{1.815cm}}
\pgfpathlineto{\pgfqpoint{0.682cm}{1.065cm}}
\pgfpathlineto{\pgfqpoint{1.246cm}{1.815cm}}
\pgfusepath{stroke}
\end{pgfscope}
\definecolor{eps2pgf_color}{gray}{0}\pgfsetstrokecolor{eps2pgf_color}\pgfsetfillcolor{eps2pgf_color}
\pgfpathmoveto{\pgfqpoint{0.273cm}{1.789cm}}
\pgfpathcurveto{\pgfqpoint{0.273cm}{1.825cm}}{\pgfqpoint{0.259cm}{1.86cm}}{\pgfqpoint{0.233cm}{1.886cm}}
\pgfpathcurveto{\pgfqpoint{0.207cm}{1.912cm}}{\pgfqpoint{0.173cm}{1.926cm}}{\pgfqpoint{0.137cm}{1.926cm}}
\pgfpathcurveto{\pgfqpoint{0.1cm}{1.926cm}}{\pgfqpoint{0.066cm}{1.912cm}}{\pgfqpoint{0.04cm}{1.886cm}}
\pgfpathcurveto{\pgfqpoint{0.014cm}{1.86cm}}{\pgfqpoint{0cm}{1.825cm}}{\pgfqpoint{0cm}{1.789cm}}
\pgfpathcurveto{\pgfqpoint{0cm}{1.753cm}}{\pgfqpoint{0.014cm}{1.718cm}}{\pgfqpoint{0.04cm}{1.692cm}}
\pgfpathcurveto{\pgfqpoint{0.066cm}{1.667cm}}{\pgfqpoint{0.1cm}{1.652cm}}{\pgfqpoint{0.137cm}{1.652cm}}
\pgfpathcurveto{\pgfqpoint{0.173cm}{1.652cm}}{\pgfqpoint{0.207cm}{1.667cm}}{\pgfqpoint{0.233cm}{1.692cm}}
\pgfpathcurveto{\pgfqpoint{0.259cm}{1.718cm}}{\pgfqpoint{0.273cm}{1.753cm}}{\pgfqpoint{0.273cm}{1.789cm}}
\pgfusepath{fill}
\pgfpathmoveto{\pgfqpoint{1.345cm}{1.765cm}}
\pgfpathcurveto{\pgfqpoint{1.345cm}{1.801cm}}{\pgfqpoint{1.331cm}{1.836cm}}{\pgfqpoint{1.305cm}{1.862cm}}
\pgfpathcurveto{\pgfqpoint{1.28cm}{1.887cm}}{\pgfqpoint{1.245cm}{1.902cm}}{\pgfqpoint{1.209cm}{1.902cm}}
\pgfpathcurveto{\pgfqpoint{1.172cm}{1.902cm}}{\pgfqpoint{1.138cm}{1.887cm}}{\pgfqpoint{1.112cm}{1.862cm}}
\pgfpathcurveto{\pgfqpoint{1.087cm}{1.836cm}}{\pgfqpoint{1.072cm}{1.801cm}}{\pgfqpoint{1.072cm}{1.765cm}}
\pgfpathcurveto{\pgfqpoint{1.072cm}{1.728cm}}{\pgfqpoint{1.087cm}{1.694cm}}{\pgfqpoint{1.112cm}{1.668cm}}
\pgfpathcurveto{\pgfqpoint{1.138cm}{1.642cm}}{\pgfqpoint{1.172cm}{1.628cm}}{\pgfqpoint{1.209cm}{1.628cm}}
\pgfpathcurveto{\pgfqpoint{1.245cm}{1.628cm}}{\pgfqpoint{1.28cm}{1.642cm}}{\pgfqpoint{1.305cm}{1.668cm}}
\pgfpathcurveto{\pgfqpoint{1.331cm}{1.694cm}}{\pgfqpoint{1.345cm}{1.728cm}}{\pgfqpoint{1.345cm}{1.765cm}}
\pgfusepath{fill}
\begin{pgfscope}
\pgfsetdash{}{0cm}
\pgfsetlinewidth{0.818mm}
\pgfsetroundcap
\pgfsetroundjoin
\pgfsetmiterlimit{7.0}
\pgfpathmoveto{\pgfqpoint{0.682cm}{1.065cm}}
\pgfpathlineto{\pgfqpoint{1.246cm}{0.315cm}}
\pgfpathlineto{\pgfqpoint{1.811cm}{1.065cm}}
\pgfusepath{stroke}
\end{pgfscope}
\pgfpathmoveto{\pgfqpoint{1.948cm}{1.065cm}}
\pgfpathcurveto{\pgfqpoint{1.948cm}{1.101cm}}{\pgfqpoint{1.933cm}{1.136cm}}{\pgfqpoint{1.907cm}{1.162cm}}
\pgfpathcurveto{\pgfqpoint{1.882cm}{1.187cm}}{\pgfqpoint{1.847cm}{1.202cm}}{\pgfqpoint{1.811cm}{1.202cm}}
\pgfpathcurveto{\pgfqpoint{1.775cm}{1.202cm}}{\pgfqpoint{1.74cm}{1.187cm}}{\pgfqpoint{1.714cm}{1.162cm}}
\pgfpathcurveto{\pgfqpoint{1.689cm}{1.136cm}}{\pgfqpoint{1.674cm}{1.101cm}}{\pgfqpoint{1.674cm}{1.065cm}}
\pgfpathcurveto{\pgfqpoint{1.674cm}{1.029cm}}{\pgfqpoint{1.689cm}{0.994cm}}{\pgfqpoint{1.714cm}{0.968cm}}
\pgfpathcurveto{\pgfqpoint{1.74cm}{0.942cm}}{\pgfqpoint{1.775cm}{0.928cm}}{\pgfqpoint{1.811cm}{0.928cm}}
\pgfpathcurveto{\pgfqpoint{1.847cm}{0.928cm}}{\pgfqpoint{1.882cm}{0.942cm}}{\pgfqpoint{1.907cm}{0.968cm}}
\pgfpathcurveto{\pgfqpoint{1.933cm}{0.994cm}}{\pgfqpoint{1.948cm}{1.029cm}}{\pgfqpoint{1.948cm}{1.065cm}}
\pgfusepath{fill}
\begin{pgfscope}
\pgfsetdash{}{0cm}
\pgfsetlinewidth{0.818mm}
\pgfsetmiterlimit{7.0}
\pgfpathmoveto{\pgfqpoint{1.246cm}{0.315cm}}
\pgfpathlineto{\pgfqpoint{1.244cm}{1.061cm}}
\pgfusepath{stroke}
\end{pgfscope}
\pgfpathmoveto{\pgfqpoint{1.38cm}{1.065cm}}
\pgfpathcurveto{\pgfqpoint{1.38cm}{1.101cm}}{\pgfqpoint{1.366cm}{1.136cm}}{\pgfqpoint{1.34cm}{1.162cm}}
\pgfpathcurveto{\pgfqpoint{1.315cm}{1.187cm}}{\pgfqpoint{1.28cm}{1.202cm}}{\pgfqpoint{1.244cm}{1.202cm}}
\pgfpathcurveto{\pgfqpoint{1.207cm}{1.202cm}}{\pgfqpoint{1.173cm}{1.187cm}}{\pgfqpoint{1.147cm}{1.162cm}}
\pgfpathcurveto{\pgfqpoint{1.121cm}{1.136cm}}{\pgfqpoint{1.107cm}{1.101cm}}{\pgfqpoint{1.107cm}{1.065cm}}
\pgfpathcurveto{\pgfqpoint{1.107cm}{1.029cm}}{\pgfqpoint{1.121cm}{0.994cm}}{\pgfqpoint{1.147cm}{0.968cm}}
\pgfpathcurveto{\pgfqpoint{1.173cm}{0.942cm}}{\pgfqpoint{1.207cm}{0.928cm}}{\pgfqpoint{1.244cm}{0.928cm}}
\pgfpathcurveto{\pgfqpoint{1.28cm}{0.928cm}}{\pgfqpoint{1.315cm}{0.942cm}}{\pgfqpoint{1.34cm}{0.968cm}}
\pgfpathcurveto{\pgfqpoint{1.366cm}{0.994cm}}{\pgfqpoint{1.38cm}{1.029cm}}{\pgfqpoint{1.38cm}{1.065cm}}
\pgfusepath{fill}
\begin{pgfscope}
\pgfsetdash{}{0cm}
\pgfsetlinewidth{0.818mm}
\pgfsetmiterlimit{4.0}
\pgfpathmoveto{\pgfqpoint{1.383cm}{0.178cm}}
\pgfpathcurveto{\pgfqpoint{1.383cm}{0.214cm}}{\pgfqpoint{1.369cm}{0.249cm}}{\pgfqpoint{1.343cm}{0.275cm}}
\pgfpathcurveto{\pgfqpoint{1.317cm}{0.3cm}}{\pgfqpoint{1.283cm}{0.315cm}}{\pgfqpoint{1.246cm}{0.315cm}}
\pgfpathcurveto{\pgfqpoint{1.21cm}{0.315cm}}{\pgfqpoint{1.175cm}{0.3cm}}{\pgfqpoint{1.15cm}{0.275cm}}
\pgfpathcurveto{\pgfqpoint{1.124cm}{0.249cm}}{\pgfqpoint{1.11cm}{0.214cm}}{\pgfqpoint{1.11cm}{0.178cm}}
\pgfpathcurveto{\pgfqpoint{1.11cm}{0.141cm}}{\pgfqpoint{1.124cm}{0.107cm}}{\pgfqpoint{1.15cm}{0.081cm}}
\pgfpathcurveto{\pgfqpoint{1.175cm}{0.055cm}}{\pgfqpoint{1.21cm}{0.041cm}}{\pgfqpoint{1.246cm}{0.041cm}}
\pgfpathcurveto{\pgfqpoint{1.283cm}{0.041cm}}{\pgfqpoint{1.317cm}{0.055cm}}{\pgfqpoint{1.343cm}{0.081cm}}
\pgfpathcurveto{\pgfqpoint{1.369cm}{0.107cm}}{\pgfqpoint{1.383cm}{0.141cm}}{\pgfqpoint{1.383cm}{0.178cm}}
\pgfusepath{stroke}
\end{pgfscope}
\end{pgfscope}
\end{pgfscope}
\end{pgfscope}
\end{tikzpicture}}} \|^{1/4}_{C_T
   \CC^{- \kappa, \varepsilon} (\rho^{\sigma})}, \quad \| X_{M,
   \varepsilon}^{\!\resizebox{!}{.8em}{
\begin{tikzpicture}
\pgfpathmoveto{\pgfqpoint{0cm}{-0.035cm}}
\pgfpathlineto{\pgfqpoint{1.976cm}{-0.035cm}}
\pgfpathlineto{\pgfqpoint{1.976cm}{1.94cm}}
\pgfpathlineto{\pgfqpoint{0cm}{1.94cm}}
\pgfpathclose
\pgfusepath{clip}
\begin{pgfscope}
\begin{pgfscope}
\pgfpathmoveto{\pgfqpoint{0cm}{-0.035cm}}
\pgfpathlineto{\pgfqpoint{1.976cm}{-0.035cm}}
\pgfpathlineto{\pgfqpoint{1.976cm}{1.94cm}}
\pgfpathlineto{\pgfqpoint{0cm}{1.94cm}}
\pgfpathclose
\pgfusepath{clip}
\begin{pgfscope}
\begin{pgfscope}
\pgfsetdash{}{0cm}
\pgfsetlinewidth{0.818mm}
\pgfsetroundcap
\pgfsetroundjoin
\pgfsetmiterlimit{7.0}
\definecolor{eps2pgf_color}{gray}{0}\pgfsetstrokecolor{eps2pgf_color}\pgfsetfillcolor{eps2pgf_color}
\pgfpathmoveto{\pgfqpoint{0.117cm}{1.815cm}}
\pgfpathlineto{\pgfqpoint{0.682cm}{1.065cm}}
\pgfpathlineto{\pgfqpoint{1.246cm}{1.815cm}}
\pgfusepath{stroke}
\end{pgfscope}
\definecolor{eps2pgf_color}{gray}{0}\pgfsetstrokecolor{eps2pgf_color}\pgfsetfillcolor{eps2pgf_color}
\pgfpathmoveto{\pgfqpoint{0.273cm}{1.789cm}}
\pgfpathcurveto{\pgfqpoint{0.273cm}{1.825cm}}{\pgfqpoint{0.259cm}{1.86cm}}{\pgfqpoint{0.233cm}{1.886cm}}
\pgfpathcurveto{\pgfqpoint{0.207cm}{1.912cm}}{\pgfqpoint{0.173cm}{1.926cm}}{\pgfqpoint{0.137cm}{1.926cm}}
\pgfpathcurveto{\pgfqpoint{0.1cm}{1.926cm}}{\pgfqpoint{0.066cm}{1.912cm}}{\pgfqpoint{0.04cm}{1.886cm}}
\pgfpathcurveto{\pgfqpoint{0.014cm}{1.86cm}}{\pgfqpoint{0cm}{1.825cm}}{\pgfqpoint{0cm}{1.789cm}}
\pgfpathcurveto{\pgfqpoint{0cm}{1.753cm}}{\pgfqpoint{0.014cm}{1.718cm}}{\pgfqpoint{0.04cm}{1.692cm}}
\pgfpathcurveto{\pgfqpoint{0.066cm}{1.667cm}}{\pgfqpoint{0.1cm}{1.652cm}}{\pgfqpoint{0.137cm}{1.652cm}}
\pgfpathcurveto{\pgfqpoint{0.173cm}{1.652cm}}{\pgfqpoint{0.207cm}{1.667cm}}{\pgfqpoint{0.233cm}{1.692cm}}
\pgfpathcurveto{\pgfqpoint{0.259cm}{1.718cm}}{\pgfqpoint{0.273cm}{1.753cm}}{\pgfqpoint{0.273cm}{1.789cm}}
\pgfusepath{fill}
\begin{pgfscope}
\pgfsetdash{}{0cm}
\pgfsetlinewidth{0.818mm}
\pgfsetmiterlimit{7.0}
\pgfpathmoveto{\pgfqpoint{0.682cm}{1.065cm}}
\pgfpathlineto{\pgfqpoint{0.679cm}{1.812cm}}
\pgfusepath{stroke}
\end{pgfscope}
\pgfpathmoveto{\pgfqpoint{0.815cm}{1.793cm}}
\pgfpathcurveto{\pgfqpoint{0.815cm}{1.829cm}}{\pgfqpoint{0.801cm}{1.864cm}}{\pgfqpoint{0.775cm}{1.89cm}}
\pgfpathcurveto{\pgfqpoint{0.75cm}{1.915cm}}{\pgfqpoint{0.715cm}{1.93cm}}{\pgfqpoint{0.679cm}{1.93cm}}
\pgfpathcurveto{\pgfqpoint{0.643cm}{1.93cm}}{\pgfqpoint{0.608cm}{1.915cm}}{\pgfqpoint{0.582cm}{1.89cm}}
\pgfpathcurveto{\pgfqpoint{0.557cm}{1.864cm}}{\pgfqpoint{0.542cm}{1.829cm}}{\pgfqpoint{0.542cm}{1.793cm}}
\pgfpathcurveto{\pgfqpoint{0.542cm}{1.756cm}}{\pgfqpoint{0.557cm}{1.722cm}}{\pgfqpoint{0.582cm}{1.696cm}}
\pgfpathcurveto{\pgfqpoint{0.608cm}{1.67cm}}{\pgfqpoint{0.643cm}{1.656cm}}{\pgfqpoint{0.679cm}{1.656cm}}
\pgfpathcurveto{\pgfqpoint{0.715cm}{1.656cm}}{\pgfqpoint{0.75cm}{1.67cm}}{\pgfqpoint{0.775cm}{1.696cm}}
\pgfpathcurveto{\pgfqpoint{0.801cm}{1.722cm}}{\pgfqpoint{0.815cm}{1.756cm}}{\pgfqpoint{0.815cm}{1.793cm}}
\pgfusepath{fill}
\pgfpathmoveto{\pgfqpoint{1.345cm}{1.765cm}}
\pgfpathcurveto{\pgfqpoint{1.345cm}{1.801cm}}{\pgfqpoint{1.331cm}{1.836cm}}{\pgfqpoint{1.305cm}{1.862cm}}
\pgfpathcurveto{\pgfqpoint{1.28cm}{1.887cm}}{\pgfqpoint{1.245cm}{1.902cm}}{\pgfqpoint{1.209cm}{1.902cm}}
\pgfpathcurveto{\pgfqpoint{1.172cm}{1.902cm}}{\pgfqpoint{1.138cm}{1.887cm}}{\pgfqpoint{1.112cm}{1.862cm}}
\pgfpathcurveto{\pgfqpoint{1.087cm}{1.836cm}}{\pgfqpoint{1.072cm}{1.801cm}}{\pgfqpoint{1.072cm}{1.765cm}}
\pgfpathcurveto{\pgfqpoint{1.072cm}{1.728cm}}{\pgfqpoint{1.087cm}{1.694cm}}{\pgfqpoint{1.112cm}{1.668cm}}
\pgfpathcurveto{\pgfqpoint{1.138cm}{1.642cm}}{\pgfqpoint{1.172cm}{1.628cm}}{\pgfqpoint{1.209cm}{1.628cm}}
\pgfpathcurveto{\pgfqpoint{1.245cm}{1.628cm}}{\pgfqpoint{1.28cm}{1.642cm}}{\pgfqpoint{1.305cm}{1.668cm}}
\pgfpathcurveto{\pgfqpoint{1.331cm}{1.694cm}}{\pgfqpoint{1.345cm}{1.728cm}}{\pgfqpoint{1.345cm}{1.765cm}}
\pgfusepath{fill}
\begin{pgfscope}
\pgfsetdash{}{0cm}
\pgfsetlinewidth{0.818mm}
\pgfsetroundcap
\pgfsetroundjoin
\pgfsetmiterlimit{7.0}
\pgfpathmoveto{\pgfqpoint{0.682cm}{1.065cm}}
\pgfpathlineto{\pgfqpoint{1.246cm}{0.315cm}}
\pgfpathlineto{\pgfqpoint{1.811cm}{1.065cm}}
\pgfusepath{stroke}
\end{pgfscope}
\pgfpathmoveto{\pgfqpoint{1.948cm}{1.065cm}}
\pgfpathcurveto{\pgfqpoint{1.948cm}{1.101cm}}{\pgfqpoint{1.933cm}{1.136cm}}{\pgfqpoint{1.907cm}{1.162cm}}
\pgfpathcurveto{\pgfqpoint{1.882cm}{1.187cm}}{\pgfqpoint{1.847cm}{1.202cm}}{\pgfqpoint{1.811cm}{1.202cm}}
\pgfpathcurveto{\pgfqpoint{1.775cm}{1.202cm}}{\pgfqpoint{1.74cm}{1.187cm}}{\pgfqpoint{1.714cm}{1.162cm}}
\pgfpathcurveto{\pgfqpoint{1.689cm}{1.136cm}}{\pgfqpoint{1.674cm}{1.101cm}}{\pgfqpoint{1.674cm}{1.065cm}}
\pgfpathcurveto{\pgfqpoint{1.674cm}{1.029cm}}{\pgfqpoint{1.689cm}{0.994cm}}{\pgfqpoint{1.714cm}{0.968cm}}
\pgfpathcurveto{\pgfqpoint{1.74cm}{0.942cm}}{\pgfqpoint{1.775cm}{0.928cm}}{\pgfqpoint{1.811cm}{0.928cm}}
\pgfpathcurveto{\pgfqpoint{1.847cm}{0.928cm}}{\pgfqpoint{1.882cm}{0.942cm}}{\pgfqpoint{1.907cm}{0.968cm}}
\pgfpathcurveto{\pgfqpoint{1.933cm}{0.994cm}}{\pgfqpoint{1.948cm}{1.029cm}}{\pgfqpoint{1.948cm}{1.065cm}}
\pgfusepath{fill}
\begin{pgfscope}
\pgfsetdash{}{0cm}
\pgfsetlinewidth{0.818mm}
\pgfsetmiterlimit{7.0}
\pgfpathmoveto{\pgfqpoint{1.246cm}{0.315cm}}
\pgfpathlineto{\pgfqpoint{1.244cm}{1.061cm}}
\pgfusepath{stroke}
\end{pgfscope}
\pgfpathmoveto{\pgfqpoint{1.38cm}{1.065cm}}
\pgfpathcurveto{\pgfqpoint{1.38cm}{1.101cm}}{\pgfqpoint{1.366cm}{1.136cm}}{\pgfqpoint{1.34cm}{1.162cm}}
\pgfpathcurveto{\pgfqpoint{1.315cm}{1.187cm}}{\pgfqpoint{1.28cm}{1.202cm}}{\pgfqpoint{1.244cm}{1.202cm}}
\pgfpathcurveto{\pgfqpoint{1.207cm}{1.202cm}}{\pgfqpoint{1.173cm}{1.187cm}}{\pgfqpoint{1.147cm}{1.162cm}}
\pgfpathcurveto{\pgfqpoint{1.121cm}{1.136cm}}{\pgfqpoint{1.107cm}{1.101cm}}{\pgfqpoint{1.107cm}{1.065cm}}
\pgfpathcurveto{\pgfqpoint{1.107cm}{1.029cm}}{\pgfqpoint{1.121cm}{0.994cm}}{\pgfqpoint{1.147cm}{0.968cm}}
\pgfpathcurveto{\pgfqpoint{1.173cm}{0.942cm}}{\pgfqpoint{1.207cm}{0.928cm}}{\pgfqpoint{1.244cm}{0.928cm}}
\pgfpathcurveto{\pgfqpoint{1.28cm}{0.928cm}}{\pgfqpoint{1.315cm}{0.942cm}}{\pgfqpoint{1.34cm}{0.968cm}}
\pgfpathcurveto{\pgfqpoint{1.366cm}{0.994cm}}{\pgfqpoint{1.38cm}{1.029cm}}{\pgfqpoint{1.38cm}{1.065cm}}
\pgfusepath{fill}
\begin{pgfscope}
\pgfsetdash{}{0cm}
\pgfsetlinewidth{0.818mm}
\pgfsetmiterlimit{4.0}
\pgfpathmoveto{\pgfqpoint{1.383cm}{0.178cm}}
\pgfpathcurveto{\pgfqpoint{1.383cm}{0.214cm}}{\pgfqpoint{1.369cm}{0.249cm}}{\pgfqpoint{1.343cm}{0.275cm}}
\pgfpathcurveto{\pgfqpoint{1.317cm}{0.3cm}}{\pgfqpoint{1.283cm}{0.315cm}}{\pgfqpoint{1.246cm}{0.315cm}}
\pgfpathcurveto{\pgfqpoint{1.21cm}{0.315cm}}{\pgfqpoint{1.175cm}{0.3cm}}{\pgfqpoint{1.15cm}{0.275cm}}
\pgfpathcurveto{\pgfqpoint{1.124cm}{0.249cm}}{\pgfqpoint{1.11cm}{0.214cm}}{\pgfqpoint{1.11cm}{0.178cm}}
\pgfpathcurveto{\pgfqpoint{1.11cm}{0.141cm}}{\pgfqpoint{1.124cm}{0.107cm}}{\pgfqpoint{1.15cm}{0.081cm}}
\pgfpathcurveto{\pgfqpoint{1.175cm}{0.055cm}}{\pgfqpoint{1.21cm}{0.041cm}}{\pgfqpoint{1.246cm}{0.041cm}}
\pgfpathcurveto{\pgfqpoint{1.283cm}{0.041cm}}{\pgfqpoint{1.317cm}{0.055cm}}{\pgfqpoint{1.343cm}{0.081cm}}
\pgfpathcurveto{\pgfqpoint{1.369cm}{0.107cm}}{\pgfqpoint{1.383cm}{0.141cm}}{\pgfqpoint{1.383cm}{0.178cm}}
\pgfusepath{stroke}
\end{pgfscope}
\end{pgfscope}
\end{pgfscope}
\end{pgfscope}
\end{tikzpicture}}} \|^{1/5}_{C_T \CC^{- 1 / 2 - \kappa, \varepsilon} (\rho^{\sigma})}.
\end{equation}
Note that it is bounded uniformly with respect to $M, \varepsilon$.  Besides, if we do not need to be precise about the exact powers, we denote by $Q_{\rho}
    (\mathbb{X}_{M, \varepsilon})$  a generic polynomial in the above
    norms of the noise terms $\mathbb{X}_{M, \varepsilon}$, whose coefficients
    depend on $\rho$ but are independent of $M,\varepsilon, \lambda$, and change from line to
    line.

\subsection{Decomposition and uniform estimates}
\label{s:estim}

With the above stochastic objects at hand, we let $\varphi_{M, \varepsilon}$
be a stationary solution to {\eqref{eq:moll}} on $\Lambda_{M, \varepsilon}$
having at each time $t \geqslant 0$ the law $\nu_{M, \varepsilon}$. We
consider its decomposition $\varphi_{M, \varepsilon} = X_{M, \varepsilon} + Y_{M, \varepsilon} +
  \phi_{M, \varepsilon}$
and deduce that $\phi_{M, \varepsilon}$ satisfies
\begin{equation}
  \begin{array}{lll}
    \LL_{\varepsilon} \phi_{M, \varepsilon} +\lambda \phi_{M, \varepsilon}^3 & = & -
    3\lambda \llbracket X_{M, \varepsilon}^2 \rrbracket \succ \phi_{M, \varepsilon} -
    3\lambda \llbracket X_{M, \varepsilon}^2 \rrbracket \preccurlyeq (Y_{M,
    \varepsilon} + \phi_{M, \varepsilon})\\
    &  & - 3\lambda^2 b_{M, \varepsilon} (X_{M, \varepsilon} + Y_{M, \varepsilon} +
    \phi_{M, \varepsilon}) - 3\lambda ( \UU^{\varepsilon}_{\leqslant L}
    \llbracket X_{M, \varepsilon}^2 \rrbracket ) \succ Y_{M,
    \varepsilon}\\
    &  & - 3\lambda X_{M, \varepsilon} (Y_{M, \varepsilon} + \phi_{M,
    \varepsilon})^2 -\lambda Y_{M, \varepsilon}^3 - 3\lambda Y_{M, \varepsilon}^2 \phi_{M,
    \varepsilon} - 3\lambda Y_{M, \varepsilon} \phi_{M, \varepsilon}^2 .
  \end{array} \label{eq:phi-eq}
\end{equation}
Our next goal is to derive energy estimates for {\eqref{eq:phi-eq}} which hold
true uniformly in both parameters $M, \varepsilon$. To this end, we recall
that all the distributions above were extended periodically to the full
lattice $\Lambda_{\varepsilon}$. Consequently, apart from the stochastic
objects, the renormalization constants and the initial conditions, all the
operations in {\eqref{eq:phi-eq}} are independent of $M$. Therefore, for
notational simplicity, we fix the parameter $M$ and omit the dependence on $M$
throughout the rest of this subsection. The following series of lemmas serves
as a preparation for our main energy estimate established in Theorem
\ref{th:energy-estimate}. Here, we make use of the approximate duality operator $D_{\rho^{4},\varepsilon}$ as well as the commutators $C_{\varepsilon},\,\tilde C_{\varepsilon}$ and $\bar C_{\varepsilon}$ introduced Section~\ref{s:l2}.

\begin{lemma}
  \label{lem:energy12}It holds
  \begin{equation}
    \frac{1}{2} \partial_t \| \rho^2 \phi_{\varepsilon} \|_{L^{2,
    \varepsilon}}^2 +\lambda \| \rho \phi_{\varepsilon} \|_{L^{4, \varepsilon}}^4 +
    m^{2} \| \rho^2 \psi_{\varepsilon} \|_{L^{2, \varepsilon}}^2 + \| \rho^2
    \nabla_{\varepsilon} \psi_{\varepsilon} \|_{L^{2, \varepsilon}}^2 =
    \Theta_{\rho^4, \varepsilon} + \Psi_{\rho^4, \varepsilon} \label{eq:en12}
  \end{equation}
  with
  \begin{equation}
    \psi_{\varepsilon} \assign \phi_{\varepsilon} + \Q_{\varepsilon}^{- 1} [3 \lambda
    \llbracket X_{\varepsilon}^2 \rrbracket \succ \phi_{\varepsilon}],
    \label{eq:psi1}
  \end{equation}
  \[ \Theta_{\rho^4, \varepsilon} \assign - \langle [\nabla_{\varepsilon}, \rho^4]
     \psi_{\varepsilon}, \nabla_{\varepsilon} \psi_{\varepsilon} \rangle_{\varepsilon} +
     \left\langle \left[ \Q_{\varepsilon}, \rho^4 \right] \Q_{\varepsilon}^{-
     1} [3\lambda \llbracket X_{\varepsilon}^2 \rrbracket \succ \phi_{\varepsilon}],
     \psi_{\varepsilon} \right\rangle_{\varepsilon} + \langle \rho^4
     \phi_{\varepsilon}^2, \lambda^2 X_{\varepsilon}^{\!\resizebox{!}{.8em}{
\begin{tikzpicture}
\pgfpathmoveto{\pgfqpoint{0cm}{-0.035cm}}
\pgfpathlineto{\pgfqpoint{1.976cm}{-0.035cm}}
\pgfpathlineto{\pgfqpoint{1.976cm}{1.94cm}}
\pgfpathlineto{\pgfqpoint{0cm}{1.94cm}}
\pgfpathclose
\pgfusepath{clip}
\begin{pgfscope}
\begin{pgfscope}
\pgfpathmoveto{\pgfqpoint{0cm}{-0.035cm}}
\pgfpathlineto{\pgfqpoint{1.976cm}{-0.035cm}}
\pgfpathlineto{\pgfqpoint{1.976cm}{1.94cm}}
\pgfpathlineto{\pgfqpoint{0cm}{1.94cm}}
\pgfpathclose
\pgfusepath{clip}
\begin{pgfscope}
\begin{pgfscope}
\pgfsetdash{}{0cm}
\pgfsetlinewidth{0.818mm}
\pgfsetroundcap
\pgfsetroundjoin
\pgfsetmiterlimit{7.0}
\definecolor{eps2pgf_color}{gray}{0}\pgfsetstrokecolor{eps2pgf_color}\pgfsetfillcolor{eps2pgf_color}
\pgfpathmoveto{\pgfqpoint{0.117cm}{1.815cm}}
\pgfpathlineto{\pgfqpoint{0.682cm}{1.065cm}}
\pgfpathlineto{\pgfqpoint{1.246cm}{1.815cm}}
\pgfusepath{stroke}
\end{pgfscope}
\definecolor{eps2pgf_color}{gray}{0}\pgfsetstrokecolor{eps2pgf_color}\pgfsetfillcolor{eps2pgf_color}
\pgfpathmoveto{\pgfqpoint{0.273cm}{1.789cm}}
\pgfpathcurveto{\pgfqpoint{0.273cm}{1.825cm}}{\pgfqpoint{0.259cm}{1.86cm}}{\pgfqpoint{0.233cm}{1.886cm}}
\pgfpathcurveto{\pgfqpoint{0.207cm}{1.912cm}}{\pgfqpoint{0.173cm}{1.926cm}}{\pgfqpoint{0.137cm}{1.926cm}}
\pgfpathcurveto{\pgfqpoint{0.1cm}{1.926cm}}{\pgfqpoint{0.066cm}{1.912cm}}{\pgfqpoint{0.04cm}{1.886cm}}
\pgfpathcurveto{\pgfqpoint{0.014cm}{1.86cm}}{\pgfqpoint{0cm}{1.825cm}}{\pgfqpoint{0cm}{1.789cm}}
\pgfpathcurveto{\pgfqpoint{0cm}{1.753cm}}{\pgfqpoint{0.014cm}{1.718cm}}{\pgfqpoint{0.04cm}{1.692cm}}
\pgfpathcurveto{\pgfqpoint{0.066cm}{1.667cm}}{\pgfqpoint{0.1cm}{1.652cm}}{\pgfqpoint{0.137cm}{1.652cm}}
\pgfpathcurveto{\pgfqpoint{0.173cm}{1.652cm}}{\pgfqpoint{0.207cm}{1.667cm}}{\pgfqpoint{0.233cm}{1.692cm}}
\pgfpathcurveto{\pgfqpoint{0.259cm}{1.718cm}}{\pgfqpoint{0.273cm}{1.753cm}}{\pgfqpoint{0.273cm}{1.789cm}}
\pgfusepath{fill}
\pgfpathmoveto{\pgfqpoint{1.345cm}{1.765cm}}
\pgfpathcurveto{\pgfqpoint{1.345cm}{1.801cm}}{\pgfqpoint{1.331cm}{1.836cm}}{\pgfqpoint{1.305cm}{1.862cm}}
\pgfpathcurveto{\pgfqpoint{1.28cm}{1.887cm}}{\pgfqpoint{1.245cm}{1.902cm}}{\pgfqpoint{1.209cm}{1.902cm}}
\pgfpathcurveto{\pgfqpoint{1.172cm}{1.902cm}}{\pgfqpoint{1.138cm}{1.887cm}}{\pgfqpoint{1.112cm}{1.862cm}}
\pgfpathcurveto{\pgfqpoint{1.087cm}{1.836cm}}{\pgfqpoint{1.072cm}{1.801cm}}{\pgfqpoint{1.072cm}{1.765cm}}
\pgfpathcurveto{\pgfqpoint{1.072cm}{1.728cm}}{\pgfqpoint{1.087cm}{1.694cm}}{\pgfqpoint{1.112cm}{1.668cm}}
\pgfpathcurveto{\pgfqpoint{1.138cm}{1.642cm}}{\pgfqpoint{1.172cm}{1.628cm}}{\pgfqpoint{1.209cm}{1.628cm}}
\pgfpathcurveto{\pgfqpoint{1.245cm}{1.628cm}}{\pgfqpoint{1.28cm}{1.642cm}}{\pgfqpoint{1.305cm}{1.668cm}}
\pgfpathcurveto{\pgfqpoint{1.331cm}{1.694cm}}{\pgfqpoint{1.345cm}{1.728cm}}{\pgfqpoint{1.345cm}{1.765cm}}
\pgfusepath{fill}
\begin{pgfscope}
\pgfsetdash{}{0cm}
\pgfsetlinewidth{0.818mm}
\pgfsetroundcap
\pgfsetroundjoin
\pgfsetmiterlimit{7.0}
\pgfpathmoveto{\pgfqpoint{0.682cm}{1.065cm}}
\pgfpathlineto{\pgfqpoint{1.246cm}{0.315cm}}
\pgfpathlineto{\pgfqpoint{1.811cm}{1.065cm}}
\pgfusepath{stroke}
\end{pgfscope}
\pgfpathmoveto{\pgfqpoint{1.948cm}{1.065cm}}
\pgfpathcurveto{\pgfqpoint{1.948cm}{1.101cm}}{\pgfqpoint{1.933cm}{1.136cm}}{\pgfqpoint{1.907cm}{1.162cm}}
\pgfpathcurveto{\pgfqpoint{1.882cm}{1.187cm}}{\pgfqpoint{1.847cm}{1.202cm}}{\pgfqpoint{1.811cm}{1.202cm}}
\pgfpathcurveto{\pgfqpoint{1.775cm}{1.202cm}}{\pgfqpoint{1.74cm}{1.187cm}}{\pgfqpoint{1.714cm}{1.162cm}}
\pgfpathcurveto{\pgfqpoint{1.689cm}{1.136cm}}{\pgfqpoint{1.674cm}{1.101cm}}{\pgfqpoint{1.674cm}{1.065cm}}
\pgfpathcurveto{\pgfqpoint{1.674cm}{1.029cm}}{\pgfqpoint{1.689cm}{0.994cm}}{\pgfqpoint{1.714cm}{0.968cm}}
\pgfpathcurveto{\pgfqpoint{1.74cm}{0.942cm}}{\pgfqpoint{1.775cm}{0.928cm}}{\pgfqpoint{1.811cm}{0.928cm}}
\pgfpathcurveto{\pgfqpoint{1.847cm}{0.928cm}}{\pgfqpoint{1.882cm}{0.942cm}}{\pgfqpoint{1.907cm}{0.968cm}}
\pgfpathcurveto{\pgfqpoint{1.933cm}{0.994cm}}{\pgfqpoint{1.948cm}{1.029cm}}{\pgfqpoint{1.948cm}{1.065cm}}
\pgfusepath{fill}
\begin{pgfscope}
\pgfsetdash{}{0cm}
\pgfsetlinewidth{0.818mm}
\pgfsetmiterlimit{7.0}
\pgfpathmoveto{\pgfqpoint{1.246cm}{0.315cm}}
\pgfpathlineto{\pgfqpoint{1.244cm}{1.061cm}}
\pgfusepath{stroke}
\end{pgfscope}
\pgfpathmoveto{\pgfqpoint{1.38cm}{1.065cm}}
\pgfpathcurveto{\pgfqpoint{1.38cm}{1.101cm}}{\pgfqpoint{1.366cm}{1.136cm}}{\pgfqpoint{1.34cm}{1.162cm}}
\pgfpathcurveto{\pgfqpoint{1.315cm}{1.187cm}}{\pgfqpoint{1.28cm}{1.202cm}}{\pgfqpoint{1.244cm}{1.202cm}}
\pgfpathcurveto{\pgfqpoint{1.207cm}{1.202cm}}{\pgfqpoint{1.173cm}{1.187cm}}{\pgfqpoint{1.147cm}{1.162cm}}
\pgfpathcurveto{\pgfqpoint{1.121cm}{1.136cm}}{\pgfqpoint{1.107cm}{1.101cm}}{\pgfqpoint{1.107cm}{1.065cm}}
\pgfpathcurveto{\pgfqpoint{1.107cm}{1.029cm}}{\pgfqpoint{1.121cm}{0.994cm}}{\pgfqpoint{1.147cm}{0.968cm}}
\pgfpathcurveto{\pgfqpoint{1.173cm}{0.942cm}}{\pgfqpoint{1.207cm}{0.928cm}}{\pgfqpoint{1.244cm}{0.928cm}}
\pgfpathcurveto{\pgfqpoint{1.28cm}{0.928cm}}{\pgfqpoint{1.315cm}{0.942cm}}{\pgfqpoint{1.34cm}{0.968cm}}
\pgfpathcurveto{\pgfqpoint{1.366cm}{0.994cm}}{\pgfqpoint{1.38cm}{1.029cm}}{\pgfqpoint{1.38cm}{1.065cm}}
\pgfusepath{fill}
\begin{pgfscope}
\pgfsetdash{}{0cm}
\pgfsetlinewidth{0.818mm}
\pgfsetmiterlimit{4.0}
\pgfpathmoveto{\pgfqpoint{1.383cm}{0.178cm}}
\pgfpathcurveto{\pgfqpoint{1.383cm}{0.214cm}}{\pgfqpoint{1.369cm}{0.249cm}}{\pgfqpoint{1.343cm}{0.275cm}}
\pgfpathcurveto{\pgfqpoint{1.317cm}{0.3cm}}{\pgfqpoint{1.283cm}{0.315cm}}{\pgfqpoint{1.246cm}{0.315cm}}
\pgfpathcurveto{\pgfqpoint{1.21cm}{0.315cm}}{\pgfqpoint{1.175cm}{0.3cm}}{\pgfqpoint{1.15cm}{0.275cm}}
\pgfpathcurveto{\pgfqpoint{1.124cm}{0.249cm}}{\pgfqpoint{1.11cm}{0.214cm}}{\pgfqpoint{1.11cm}{0.178cm}}
\pgfpathcurveto{\pgfqpoint{1.11cm}{0.141cm}}{\pgfqpoint{1.124cm}{0.107cm}}{\pgfqpoint{1.15cm}{0.081cm}}
\pgfpathcurveto{\pgfqpoint{1.175cm}{0.055cm}}{\pgfqpoint{1.21cm}{0.041cm}}{\pgfqpoint{1.246cm}{0.041cm}}
\pgfpathcurveto{\pgfqpoint{1.283cm}{0.041cm}}{\pgfqpoint{1.317cm}{0.055cm}}{\pgfqpoint{1.343cm}{0.081cm}}
\pgfpathcurveto{\pgfqpoint{1.369cm}{0.107cm}}{\pgfqpoint{1.383cm}{0.141cm}}{\pgfqpoint{1.383cm}{0.178cm}}
\pgfusepath{stroke}
\end{pgfscope}
\end{pgfscope}
\end{pgfscope}
\end{pgfscope}
\end{tikzpicture}}} \rangle_{\varepsilon} \]
  \[ + D_{\rho^4, \varepsilon} (\phi_{\varepsilon}, - 3\lambda \llbracket
     X_{\varepsilon}^2 \rrbracket, \phi_{\varepsilon}) + \langle \rho^4
     \phi_{\varepsilon}, \tilde{C}_{\varepsilon} (\phi_{\varepsilon}, 3\lambda
     \llbracket X_{\varepsilon}^2 \rrbracket, 3\lambda \llbracket X_{\varepsilon}^2
     \rrbracket) \rangle_{\varepsilon} \]
  \[ + D_{\rho^4, \varepsilon} \left( \phi_{\varepsilon}, 3\lambda \llbracket
     X_{\varepsilon}^2 \rrbracket, \Q_{\varepsilon}^{- 1} [3\lambda \llbracket
     X_{\varepsilon}^2 \rrbracket \succ \phi_{\varepsilon}] \right), \]
  \[ \Psi_{\rho^4, \varepsilon} \assign \langle \rho^4 \phi_{\varepsilon}, - 3\lambda
     \llbracket X_{\varepsilon}^2 \rrbracket \prec (Y_{\varepsilon} +
     \phi_{\varepsilon}) - 3\lambda X_{\varepsilon} (Y_{\varepsilon} +
     \phi_{\varepsilon})^2 - \lambda Y_{\varepsilon}^3 - 3 \lambda Y_{\varepsilon}^2
     \phi_{\varepsilon} - 3\lambda Y_{\varepsilon} \phi_{\varepsilon}^2 \rangle_{\varepsilon} \]
  \[ + \langle \rho^4 \phi_{\varepsilon}, - 3\lambda ( \UU^{\varepsilon}_{\leqslant}
     \llbracket X_{\varepsilon}^2 \rrbracket ) \succ Y_{\varepsilon} + \lambda^2
     Z_{\varepsilon} \rangle_{\varepsilon}, \]
  and
  \begin{equation}\label{eq:def-Z}
   Z_{\varepsilon} \assign X_{\varepsilon}^{\!\resizebox{!}{.8em}{
\begin{tikzpicture}
\pgfpathmoveto{\pgfqpoint{0cm}{-0.035cm}}
\pgfpathlineto{\pgfqpoint{1.976cm}{-0.035cm}}
\pgfpathlineto{\pgfqpoint{1.976cm}{1.94cm}}
\pgfpathlineto{\pgfqpoint{0cm}{1.94cm}}
\pgfpathclose
\pgfusepath{clip}
\begin{pgfscope}
\begin{pgfscope}
\pgfpathmoveto{\pgfqpoint{0cm}{-0.035cm}}
\pgfpathlineto{\pgfqpoint{1.976cm}{-0.035cm}}
\pgfpathlineto{\pgfqpoint{1.976cm}{1.94cm}}
\pgfpathlineto{\pgfqpoint{0cm}{1.94cm}}
\pgfpathclose
\pgfusepath{clip}
\begin{pgfscope}
\begin{pgfscope}
\pgfsetdash{}{0cm}
\pgfsetlinewidth{0.818mm}
\pgfsetroundcap
\pgfsetroundjoin
\pgfsetmiterlimit{7.0}
\definecolor{eps2pgf_color}{gray}{0}\pgfsetstrokecolor{eps2pgf_color}\pgfsetfillcolor{eps2pgf_color}
\pgfpathmoveto{\pgfqpoint{0.117cm}{1.815cm}}
\pgfpathlineto{\pgfqpoint{0.682cm}{1.065cm}}
\pgfpathlineto{\pgfqpoint{1.246cm}{1.815cm}}
\pgfusepath{stroke}
\end{pgfscope}
\definecolor{eps2pgf_color}{gray}{0}\pgfsetstrokecolor{eps2pgf_color}\pgfsetfillcolor{eps2pgf_color}
\pgfpathmoveto{\pgfqpoint{0.273cm}{1.789cm}}
\pgfpathcurveto{\pgfqpoint{0.273cm}{1.825cm}}{\pgfqpoint{0.259cm}{1.86cm}}{\pgfqpoint{0.233cm}{1.886cm}}
\pgfpathcurveto{\pgfqpoint{0.207cm}{1.912cm}}{\pgfqpoint{0.173cm}{1.926cm}}{\pgfqpoint{0.137cm}{1.926cm}}
\pgfpathcurveto{\pgfqpoint{0.1cm}{1.926cm}}{\pgfqpoint{0.066cm}{1.912cm}}{\pgfqpoint{0.04cm}{1.886cm}}
\pgfpathcurveto{\pgfqpoint{0.014cm}{1.86cm}}{\pgfqpoint{0cm}{1.825cm}}{\pgfqpoint{0cm}{1.789cm}}
\pgfpathcurveto{\pgfqpoint{0cm}{1.753cm}}{\pgfqpoint{0.014cm}{1.718cm}}{\pgfqpoint{0.04cm}{1.692cm}}
\pgfpathcurveto{\pgfqpoint{0.066cm}{1.667cm}}{\pgfqpoint{0.1cm}{1.652cm}}{\pgfqpoint{0.137cm}{1.652cm}}
\pgfpathcurveto{\pgfqpoint{0.173cm}{1.652cm}}{\pgfqpoint{0.207cm}{1.667cm}}{\pgfqpoint{0.233cm}{1.692cm}}
\pgfpathcurveto{\pgfqpoint{0.259cm}{1.718cm}}{\pgfqpoint{0.273cm}{1.753cm}}{\pgfqpoint{0.273cm}{1.789cm}}
\pgfusepath{fill}
\begin{pgfscope}
\pgfsetdash{}{0cm}
\pgfsetlinewidth{0.818mm}
\pgfsetmiterlimit{7.0}
\pgfpathmoveto{\pgfqpoint{0.682cm}{1.065cm}}
\pgfpathlineto{\pgfqpoint{0.679cm}{1.812cm}}
\pgfusepath{stroke}
\end{pgfscope}
\pgfpathmoveto{\pgfqpoint{0.815cm}{1.793cm}}
\pgfpathcurveto{\pgfqpoint{0.815cm}{1.829cm}}{\pgfqpoint{0.801cm}{1.864cm}}{\pgfqpoint{0.775cm}{1.89cm}}
\pgfpathcurveto{\pgfqpoint{0.75cm}{1.915cm}}{\pgfqpoint{0.715cm}{1.93cm}}{\pgfqpoint{0.679cm}{1.93cm}}
\pgfpathcurveto{\pgfqpoint{0.643cm}{1.93cm}}{\pgfqpoint{0.608cm}{1.915cm}}{\pgfqpoint{0.582cm}{1.89cm}}
\pgfpathcurveto{\pgfqpoint{0.557cm}{1.864cm}}{\pgfqpoint{0.542cm}{1.829cm}}{\pgfqpoint{0.542cm}{1.793cm}}
\pgfpathcurveto{\pgfqpoint{0.542cm}{1.756cm}}{\pgfqpoint{0.557cm}{1.722cm}}{\pgfqpoint{0.582cm}{1.696cm}}
\pgfpathcurveto{\pgfqpoint{0.608cm}{1.67cm}}{\pgfqpoint{0.643cm}{1.656cm}}{\pgfqpoint{0.679cm}{1.656cm}}
\pgfpathcurveto{\pgfqpoint{0.715cm}{1.656cm}}{\pgfqpoint{0.75cm}{1.67cm}}{\pgfqpoint{0.775cm}{1.696cm}}
\pgfpathcurveto{\pgfqpoint{0.801cm}{1.722cm}}{\pgfqpoint{0.815cm}{1.756cm}}{\pgfqpoint{0.815cm}{1.793cm}}
\pgfusepath{fill}
\pgfpathmoveto{\pgfqpoint{1.345cm}{1.765cm}}
\pgfpathcurveto{\pgfqpoint{1.345cm}{1.801cm}}{\pgfqpoint{1.331cm}{1.836cm}}{\pgfqpoint{1.305cm}{1.862cm}}
\pgfpathcurveto{\pgfqpoint{1.28cm}{1.887cm}}{\pgfqpoint{1.245cm}{1.902cm}}{\pgfqpoint{1.209cm}{1.902cm}}
\pgfpathcurveto{\pgfqpoint{1.172cm}{1.902cm}}{\pgfqpoint{1.138cm}{1.887cm}}{\pgfqpoint{1.112cm}{1.862cm}}
\pgfpathcurveto{\pgfqpoint{1.087cm}{1.836cm}}{\pgfqpoint{1.072cm}{1.801cm}}{\pgfqpoint{1.072cm}{1.765cm}}
\pgfpathcurveto{\pgfqpoint{1.072cm}{1.728cm}}{\pgfqpoint{1.087cm}{1.694cm}}{\pgfqpoint{1.112cm}{1.668cm}}
\pgfpathcurveto{\pgfqpoint{1.138cm}{1.642cm}}{\pgfqpoint{1.172cm}{1.628cm}}{\pgfqpoint{1.209cm}{1.628cm}}
\pgfpathcurveto{\pgfqpoint{1.245cm}{1.628cm}}{\pgfqpoint{1.28cm}{1.642cm}}{\pgfqpoint{1.305cm}{1.668cm}}
\pgfpathcurveto{\pgfqpoint{1.331cm}{1.694cm}}{\pgfqpoint{1.345cm}{1.728cm}}{\pgfqpoint{1.345cm}{1.765cm}}
\pgfusepath{fill}
\begin{pgfscope}
\pgfsetdash{}{0cm}
\pgfsetlinewidth{0.818mm}
\pgfsetroundcap
\pgfsetroundjoin
\pgfsetmiterlimit{7.0}
\pgfpathmoveto{\pgfqpoint{0.682cm}{1.065cm}}
\pgfpathlineto{\pgfqpoint{1.246cm}{0.315cm}}
\pgfpathlineto{\pgfqpoint{1.811cm}{1.065cm}}
\pgfusepath{stroke}
\end{pgfscope}
\pgfpathmoveto{\pgfqpoint{1.948cm}{1.065cm}}
\pgfpathcurveto{\pgfqpoint{1.948cm}{1.101cm}}{\pgfqpoint{1.933cm}{1.136cm}}{\pgfqpoint{1.907cm}{1.162cm}}
\pgfpathcurveto{\pgfqpoint{1.882cm}{1.187cm}}{\pgfqpoint{1.847cm}{1.202cm}}{\pgfqpoint{1.811cm}{1.202cm}}
\pgfpathcurveto{\pgfqpoint{1.775cm}{1.202cm}}{\pgfqpoint{1.74cm}{1.187cm}}{\pgfqpoint{1.714cm}{1.162cm}}
\pgfpathcurveto{\pgfqpoint{1.689cm}{1.136cm}}{\pgfqpoint{1.674cm}{1.101cm}}{\pgfqpoint{1.674cm}{1.065cm}}
\pgfpathcurveto{\pgfqpoint{1.674cm}{1.029cm}}{\pgfqpoint{1.689cm}{0.994cm}}{\pgfqpoint{1.714cm}{0.968cm}}
\pgfpathcurveto{\pgfqpoint{1.74cm}{0.942cm}}{\pgfqpoint{1.775cm}{0.928cm}}{\pgfqpoint{1.811cm}{0.928cm}}
\pgfpathcurveto{\pgfqpoint{1.847cm}{0.928cm}}{\pgfqpoint{1.882cm}{0.942cm}}{\pgfqpoint{1.907cm}{0.968cm}}
\pgfpathcurveto{\pgfqpoint{1.933cm}{0.994cm}}{\pgfqpoint{1.948cm}{1.029cm}}{\pgfqpoint{1.948cm}{1.065cm}}
\pgfusepath{fill}
\begin{pgfscope}
\pgfsetdash{}{0cm}
\pgfsetlinewidth{0.818mm}
\pgfsetmiterlimit{7.0}
\pgfpathmoveto{\pgfqpoint{1.246cm}{0.315cm}}
\pgfpathlineto{\pgfqpoint{1.244cm}{1.061cm}}
\pgfusepath{stroke}
\end{pgfscope}
\pgfpathmoveto{\pgfqpoint{1.38cm}{1.065cm}}
\pgfpathcurveto{\pgfqpoint{1.38cm}{1.101cm}}{\pgfqpoint{1.366cm}{1.136cm}}{\pgfqpoint{1.34cm}{1.162cm}}
\pgfpathcurveto{\pgfqpoint{1.315cm}{1.187cm}}{\pgfqpoint{1.28cm}{1.202cm}}{\pgfqpoint{1.244cm}{1.202cm}}
\pgfpathcurveto{\pgfqpoint{1.207cm}{1.202cm}}{\pgfqpoint{1.173cm}{1.187cm}}{\pgfqpoint{1.147cm}{1.162cm}}
\pgfpathcurveto{\pgfqpoint{1.121cm}{1.136cm}}{\pgfqpoint{1.107cm}{1.101cm}}{\pgfqpoint{1.107cm}{1.065cm}}
\pgfpathcurveto{\pgfqpoint{1.107cm}{1.029cm}}{\pgfqpoint{1.121cm}{0.994cm}}{\pgfqpoint{1.147cm}{0.968cm}}
\pgfpathcurveto{\pgfqpoint{1.173cm}{0.942cm}}{\pgfqpoint{1.207cm}{0.928cm}}{\pgfqpoint{1.244cm}{0.928cm}}
\pgfpathcurveto{\pgfqpoint{1.28cm}{0.928cm}}{\pgfqpoint{1.315cm}{0.942cm}}{\pgfqpoint{1.34cm}{0.968cm}}
\pgfpathcurveto{\pgfqpoint{1.366cm}{0.994cm}}{\pgfqpoint{1.38cm}{1.029cm}}{\pgfqpoint{1.38cm}{1.065cm}}
\pgfusepath{fill}
\begin{pgfscope}
\pgfsetdash{}{0cm}
\pgfsetlinewidth{0.818mm}
\pgfsetmiterlimit{4.0}
\pgfpathmoveto{\pgfqpoint{1.383cm}{0.178cm}}
\pgfpathcurveto{\pgfqpoint{1.383cm}{0.214cm}}{\pgfqpoint{1.369cm}{0.249cm}}{\pgfqpoint{1.343cm}{0.275cm}}
\pgfpathcurveto{\pgfqpoint{1.317cm}{0.3cm}}{\pgfqpoint{1.283cm}{0.315cm}}{\pgfqpoint{1.246cm}{0.315cm}}
\pgfpathcurveto{\pgfqpoint{1.21cm}{0.315cm}}{\pgfqpoint{1.175cm}{0.3cm}}{\pgfqpoint{1.15cm}{0.275cm}}
\pgfpathcurveto{\pgfqpoint{1.124cm}{0.249cm}}{\pgfqpoint{1.11cm}{0.214cm}}{\pgfqpoint{1.11cm}{0.178cm}}
\pgfpathcurveto{\pgfqpoint{1.11cm}{0.141cm}}{\pgfqpoint{1.124cm}{0.107cm}}{\pgfqpoint{1.15cm}{0.081cm}}
\pgfpathcurveto{\pgfqpoint{1.175cm}{0.055cm}}{\pgfqpoint{1.21cm}{0.041cm}}{\pgfqpoint{1.246cm}{0.041cm}}
\pgfpathcurveto{\pgfqpoint{1.283cm}{0.041cm}}{\pgfqpoint{1.317cm}{0.055cm}}{\pgfqpoint{1.343cm}{0.081cm}}
\pgfpathcurveto{\pgfqpoint{1.369cm}{0.107cm}}{\pgfqpoint{1.383cm}{0.141cm}}{\pgfqpoint{1.383cm}{0.178cm}}
\pgfusepath{stroke}
\end{pgfscope}
\end{pgfscope}
\end{pgfscope}
\end{pgfscope}
\end{tikzpicture}}} +
     \tilde{X}_{\varepsilon}^{\!\resizebox{!}{.8em}{
\begin{tikzpicture}
\pgfpathmoveto{\pgfqpoint{0cm}{-0.035cm}}
\pgfpathlineto{\pgfqpoint{1.976cm}{-0.035cm}}
\pgfpathlineto{\pgfqpoint{1.976cm}{1.94cm}}
\pgfpathlineto{\pgfqpoint{0cm}{1.94cm}}
\pgfpathclose
\pgfusepath{clip}
\begin{pgfscope}
\begin{pgfscope}
\pgfpathmoveto{\pgfqpoint{0cm}{-0.035cm}}
\pgfpathlineto{\pgfqpoint{1.976cm}{-0.035cm}}
\pgfpathlineto{\pgfqpoint{1.976cm}{1.94cm}}
\pgfpathlineto{\pgfqpoint{0cm}{1.94cm}}
\pgfpathclose
\pgfusepath{clip}
\begin{pgfscope}
\begin{pgfscope}
\pgfsetdash{}{0cm}
\pgfsetlinewidth{0.818mm}
\pgfsetroundcap
\pgfsetroundjoin
\pgfsetmiterlimit{7.0}
\definecolor{eps2pgf_color}{gray}{0}\pgfsetstrokecolor{eps2pgf_color}\pgfsetfillcolor{eps2pgf_color}
\pgfpathmoveto{\pgfqpoint{0.117cm}{1.815cm}}
\pgfpathlineto{\pgfqpoint{0.682cm}{1.065cm}}
\pgfpathlineto{\pgfqpoint{1.246cm}{1.815cm}}
\pgfusepath{stroke}
\end{pgfscope}
\definecolor{eps2pgf_color}{gray}{0}\pgfsetstrokecolor{eps2pgf_color}\pgfsetfillcolor{eps2pgf_color}
\pgfpathmoveto{\pgfqpoint{0.273cm}{1.789cm}}
\pgfpathcurveto{\pgfqpoint{0.273cm}{1.825cm}}{\pgfqpoint{0.259cm}{1.86cm}}{\pgfqpoint{0.233cm}{1.886cm}}
\pgfpathcurveto{\pgfqpoint{0.207cm}{1.912cm}}{\pgfqpoint{0.173cm}{1.926cm}}{\pgfqpoint{0.137cm}{1.926cm}}
\pgfpathcurveto{\pgfqpoint{0.1cm}{1.926cm}}{\pgfqpoint{0.066cm}{1.912cm}}{\pgfqpoint{0.04cm}{1.886cm}}
\pgfpathcurveto{\pgfqpoint{0.014cm}{1.86cm}}{\pgfqpoint{0cm}{1.825cm}}{\pgfqpoint{0cm}{1.789cm}}
\pgfpathcurveto{\pgfqpoint{0cm}{1.753cm}}{\pgfqpoint{0.014cm}{1.718cm}}{\pgfqpoint{0.04cm}{1.692cm}}
\pgfpathcurveto{\pgfqpoint{0.066cm}{1.667cm}}{\pgfqpoint{0.1cm}{1.652cm}}{\pgfqpoint{0.137cm}{1.652cm}}
\pgfpathcurveto{\pgfqpoint{0.173cm}{1.652cm}}{\pgfqpoint{0.207cm}{1.667cm}}{\pgfqpoint{0.233cm}{1.692cm}}
\pgfpathcurveto{\pgfqpoint{0.259cm}{1.718cm}}{\pgfqpoint{0.273cm}{1.753cm}}{\pgfqpoint{0.273cm}{1.789cm}}
\pgfusepath{fill}
\pgfpathmoveto{\pgfqpoint{1.345cm}{1.765cm}}
\pgfpathcurveto{\pgfqpoint{1.345cm}{1.801cm}}{\pgfqpoint{1.331cm}{1.836cm}}{\pgfqpoint{1.305cm}{1.862cm}}
\pgfpathcurveto{\pgfqpoint{1.28cm}{1.887cm}}{\pgfqpoint{1.245cm}{1.902cm}}{\pgfqpoint{1.209cm}{1.902cm}}
\pgfpathcurveto{\pgfqpoint{1.172cm}{1.902cm}}{\pgfqpoint{1.138cm}{1.887cm}}{\pgfqpoint{1.112cm}{1.862cm}}
\pgfpathcurveto{\pgfqpoint{1.087cm}{1.836cm}}{\pgfqpoint{1.072cm}{1.801cm}}{\pgfqpoint{1.072cm}{1.765cm}}
\pgfpathcurveto{\pgfqpoint{1.072cm}{1.728cm}}{\pgfqpoint{1.087cm}{1.694cm}}{\pgfqpoint{1.112cm}{1.668cm}}
\pgfpathcurveto{\pgfqpoint{1.138cm}{1.642cm}}{\pgfqpoint{1.172cm}{1.628cm}}{\pgfqpoint{1.209cm}{1.628cm}}
\pgfpathcurveto{\pgfqpoint{1.245cm}{1.628cm}}{\pgfqpoint{1.28cm}{1.642cm}}{\pgfqpoint{1.305cm}{1.668cm}}
\pgfpathcurveto{\pgfqpoint{1.331cm}{1.694cm}}{\pgfqpoint{1.345cm}{1.728cm}}{\pgfqpoint{1.345cm}{1.765cm}}
\pgfusepath{fill}
\begin{pgfscope}
\pgfsetdash{}{0cm}
\pgfsetlinewidth{0.818mm}
\pgfsetroundcap
\pgfsetroundjoin
\pgfsetmiterlimit{7.0}
\pgfpathmoveto{\pgfqpoint{0.682cm}{1.065cm}}
\pgfpathlineto{\pgfqpoint{1.246cm}{0.315cm}}
\pgfpathlineto{\pgfqpoint{1.811cm}{1.065cm}}
\pgfusepath{stroke}
\end{pgfscope}
\pgfpathmoveto{\pgfqpoint{1.948cm}{1.065cm}}
\pgfpathcurveto{\pgfqpoint{1.948cm}{1.101cm}}{\pgfqpoint{1.933cm}{1.136cm}}{\pgfqpoint{1.907cm}{1.162cm}}
\pgfpathcurveto{\pgfqpoint{1.882cm}{1.187cm}}{\pgfqpoint{1.847cm}{1.202cm}}{\pgfqpoint{1.811cm}{1.202cm}}
\pgfpathcurveto{\pgfqpoint{1.775cm}{1.202cm}}{\pgfqpoint{1.74cm}{1.187cm}}{\pgfqpoint{1.714cm}{1.162cm}}
\pgfpathcurveto{\pgfqpoint{1.689cm}{1.136cm}}{\pgfqpoint{1.674cm}{1.101cm}}{\pgfqpoint{1.674cm}{1.065cm}}
\pgfpathcurveto{\pgfqpoint{1.674cm}{1.029cm}}{\pgfqpoint{1.689cm}{0.994cm}}{\pgfqpoint{1.714cm}{0.968cm}}
\pgfpathcurveto{\pgfqpoint{1.74cm}{0.942cm}}{\pgfqpoint{1.775cm}{0.928cm}}{\pgfqpoint{1.811cm}{0.928cm}}
\pgfpathcurveto{\pgfqpoint{1.847cm}{0.928cm}}{\pgfqpoint{1.882cm}{0.942cm}}{\pgfqpoint{1.907cm}{0.968cm}}
\pgfpathcurveto{\pgfqpoint{1.933cm}{0.994cm}}{\pgfqpoint{1.948cm}{1.029cm}}{\pgfqpoint{1.948cm}{1.065cm}}
\pgfusepath{fill}
\begin{pgfscope}
\pgfsetdash{}{0cm}
\pgfsetlinewidth{0.818mm}
\pgfsetmiterlimit{7.0}
\pgfpathmoveto{\pgfqpoint{1.246cm}{0.315cm}}
\pgfpathlineto{\pgfqpoint{1.244cm}{1.061cm}}
\pgfusepath{stroke}
\end{pgfscope}
\pgfpathmoveto{\pgfqpoint{1.38cm}{1.065cm}}
\pgfpathcurveto{\pgfqpoint{1.38cm}{1.101cm}}{\pgfqpoint{1.366cm}{1.136cm}}{\pgfqpoint{1.34cm}{1.162cm}}
\pgfpathcurveto{\pgfqpoint{1.315cm}{1.187cm}}{\pgfqpoint{1.28cm}{1.202cm}}{\pgfqpoint{1.244cm}{1.202cm}}
\pgfpathcurveto{\pgfqpoint{1.207cm}{1.202cm}}{\pgfqpoint{1.173cm}{1.187cm}}{\pgfqpoint{1.147cm}{1.162cm}}
\pgfpathcurveto{\pgfqpoint{1.121cm}{1.136cm}}{\pgfqpoint{1.107cm}{1.101cm}}{\pgfqpoint{1.107cm}{1.065cm}}
\pgfpathcurveto{\pgfqpoint{1.107cm}{1.029cm}}{\pgfqpoint{1.121cm}{0.994cm}}{\pgfqpoint{1.147cm}{0.968cm}}
\pgfpathcurveto{\pgfqpoint{1.173cm}{0.942cm}}{\pgfqpoint{1.207cm}{0.928cm}}{\pgfqpoint{1.244cm}{0.928cm}}
\pgfpathcurveto{\pgfqpoint{1.28cm}{0.928cm}}{\pgfqpoint{1.315cm}{0.942cm}}{\pgfqpoint{1.34cm}{0.968cm}}
\pgfpathcurveto{\pgfqpoint{1.366cm}{0.994cm}}{\pgfqpoint{1.38cm}{1.029cm}}{\pgfqpoint{1.38cm}{1.065cm}}
\pgfusepath{fill}
\begin{pgfscope}
\pgfsetdash{}{0cm}
\pgfsetlinewidth{0.818mm}
\pgfsetmiterlimit{4.0}
\pgfpathmoveto{\pgfqpoint{1.383cm}{0.178cm}}
\pgfpathcurveto{\pgfqpoint{1.383cm}{0.214cm}}{\pgfqpoint{1.369cm}{0.249cm}}{\pgfqpoint{1.343cm}{0.275cm}}
\pgfpathcurveto{\pgfqpoint{1.317cm}{0.3cm}}{\pgfqpoint{1.283cm}{0.315cm}}{\pgfqpoint{1.246cm}{0.315cm}}
\pgfpathcurveto{\pgfqpoint{1.21cm}{0.315cm}}{\pgfqpoint{1.175cm}{0.3cm}}{\pgfqpoint{1.15cm}{0.275cm}}
\pgfpathcurveto{\pgfqpoint{1.124cm}{0.249cm}}{\pgfqpoint{1.11cm}{0.214cm}}{\pgfqpoint{1.11cm}{0.178cm}}
\pgfpathcurveto{\pgfqpoint{1.11cm}{0.141cm}}{\pgfqpoint{1.124cm}{0.107cm}}{\pgfqpoint{1.15cm}{0.081cm}}
\pgfpathcurveto{\pgfqpoint{1.175cm}{0.055cm}}{\pgfqpoint{1.21cm}{0.041cm}}{\pgfqpoint{1.246cm}{0.041cm}}
\pgfpathcurveto{\pgfqpoint{1.283cm}{0.041cm}}{\pgfqpoint{1.317cm}{0.055cm}}{\pgfqpoint{1.343cm}{0.081cm}}
\pgfpathcurveto{\pgfqpoint{1.369cm}{0.107cm}}{\pgfqpoint{1.383cm}{0.141cm}}{\pgfqpoint{1.383cm}{0.178cm}}
\pgfusepath{stroke}
\end{pgfscope}
\end{pgfscope}
\end{pgfscope}
\end{pgfscope}
\end{tikzpicture}}} Y_{\varepsilon} + 3
     (\tilde{b}_{\varepsilon} - b_{\varepsilon}) Y_{\varepsilon} +
     \bar{C}_{\varepsilon} (Y_{\varepsilon}, 3 \llbracket X_{\varepsilon}^2
     \rrbracket, 3 \llbracket X_{\varepsilon}^2 \rrbracket) - 3 \llbracket
     X_{\varepsilon}^2 \rrbracket \circ \LL_{\varepsilon}^{- 1} \left( 3
     \UU^{\varepsilon}_{\leqslant} \llbracket X_{\varepsilon}^2 \rrbracket
     \succ Y_{\varepsilon} \right) .
     \end{equation}
\end{lemma}

\begin{proof}
  Noting that {\eqref{eq:phi-eq}} is of the form $\LL_{\varepsilon}
  \phi_{\varepsilon} +\lambda \phi_{\varepsilon}^3 = U_{\varepsilon}$, we may test
  this equation by $\rho^4 \phi_{\varepsilon}$ to deduce
  \[ \frac{1}{2} \partial_t \langle \rho^2 \phi_{\varepsilon}, \rho^2
     \phi_{\varepsilon} \rangle_{\varepsilon} + \lambda \langle \rho^2 \phi_{\varepsilon}, \rho^2
     \phi_{\varepsilon}^3 \rangle_{\varepsilon} = \Phi_{\rho^4, \varepsilon} + \Psi_{\rho^4,
     \varepsilon}, \]
  with
  \[ \Phi_{\rho^4, \varepsilon} \assign \langle \rho^4 \phi_{\varepsilon}, -
     \Q_{\varepsilon} \phi_{\varepsilon} - 3\lambda \llbracket X_{\varepsilon}^2
     \rrbracket \succ \phi_{\varepsilon} - 3 \lambda\llbracket X_{\varepsilon}^2
     \rrbracket \circ \phi_{\varepsilon} - 3\lambda^{2} b_{\varepsilon}
     \phi_{\varepsilon} \rangle_{\varepsilon}, \]
  and
  \[ \Psi_{\rho^4, \varepsilon} \assign \langle \rho^4 \phi_{\varepsilon}, - 3\lambda
     \llbracket X_{\varepsilon}^2 \rrbracket \prec (Y_{\varepsilon} +
     \phi_{\varepsilon}) - 3\lambda X_{\varepsilon} (Y_{\varepsilon} +
     \phi_{\varepsilon})^2 -\lambda Y_{\varepsilon}^3 - 3\lambda Y_{\varepsilon}^2
     \phi_{\varepsilon} - 3\lambda Y_{\varepsilon} \phi_{\varepsilon}^2 \rangle_{\varepsilon} \]
  \[ + \langle \rho^4 \phi_{\varepsilon}, - 3\lambda \left(
     \UU^{\varepsilon}_{\leqslant L} \llbracket X_{\varepsilon}^2 \rrbracket
     \right) \succ Y_{\varepsilon} - 3\lambda \llbracket X_{\varepsilon}^2 \rrbracket
     \circ Y_{\varepsilon} - 3\lambda^2 b_{\varepsilon} (X_{\varepsilon} +
     Y_{\varepsilon}) \rangle_{\varepsilon} . \]
  We use the fact that $(f \succ)$ is an approximate adjoint to $(f \circ)$
  according to Lemma~\ref{lem:dual1} to rewrite the resonant term as
  \[ \langle \rho^4 \phi_{\varepsilon}, - 3\lambda \llbracket X_{\varepsilon}^2
     \rrbracket \circ \phi_{\varepsilon} \rangle_{\varepsilon} = \langle \rho^4
     \phi_{\varepsilon}, - 3\lambda \llbracket X_{\varepsilon}^2 \rrbracket \succ
     \phi_{\varepsilon} \rangle_{\varepsilon} + D_{\rho^4, \varepsilon} (\phi_{\varepsilon},
     - 3\lambda \llbracket X_{\varepsilon}^2 \rrbracket, \phi_{\varepsilon}), \]
  and use the definition of $\psi$ in {\eqref{eq:psi1}} to rewrite
  $\Phi_{\rho, \varepsilon}$ as
  \[ \Phi_{\rho^4, \varepsilon} = \langle \rho^4 \psi_{\varepsilon}, -
     \Q_{\varepsilon} \psi_{\varepsilon} \rangle_{\varepsilon} + \left\langle \left[
     \Q_{\varepsilon}, \rho^4 \right] \Q_{\varepsilon}^{- 1} [3\lambda \llbracket
     X_{\varepsilon}^2 \rrbracket \succ \phi_{\varepsilon}],
     \psi_{\varepsilon} \right\rangle_{\varepsilon} \]
  \[ + \langle \rho^4 [3\lambda \llbracket X_{\varepsilon}^2 \rrbracket \succ
     \phi_{\varepsilon}], \Q_{\varepsilon}^{- 1} [3\lambda \llbracket
     X_{\varepsilon}^2 \rrbracket \succ \phi_{\varepsilon}] \rangle_{\varepsilon} - 3\lambda^2
     b_{\varepsilon} \langle \rho^4 \phi_{\varepsilon}, \phi_{\varepsilon}
     \rangle_{\varepsilon} + D_{\rho^4, \varepsilon} (\phi_{\varepsilon}, - 3 \lambda\llbracket
     X_{\varepsilon}^2 \rrbracket, \phi_{\varepsilon}) . \]
  For the first term we write
  \[ \langle \rho^4 \psi_{\varepsilon}, - \Q_{\varepsilon} \psi_{\varepsilon}
     \rangle_{\varepsilon} = - m^{2} \langle \rho^4 \psi_{\varepsilon}, \psi_{\varepsilon}
     \rangle_{\varepsilon} - \langle \rho^4 \nabla_{\varepsilon} \psi_{\varepsilon},
     \nabla_{\varepsilon} \psi_{\varepsilon} \rangle_{\varepsilon} - \langle
     [\nabla_{\varepsilon}, \rho^4] \psi_{\varepsilon}, \nabla_{\varepsilon}
     \psi_{\varepsilon} \rangle_{\varepsilon} . \]
  Next, we use again Lemma~\ref{lem:dual1} to simplify the quadratic term as
  \[ \langle \rho^4 [3\lambda \llbracket X_{\varepsilon}^2 \rrbracket \succ
     \phi_{\varepsilon}], \Q_{\varepsilon}^{- 1} [3\lambda \llbracket
     X_{\varepsilon}^2 \rrbracket \succ \phi_{\varepsilon}] \rangle_{\varepsilon} =
     \left\langle \rho^4 \phi_{\varepsilon}, 3\lambda \llbracket X_{\varepsilon}^2
     \rrbracket \circ \Q_{\varepsilon}^{- 1} [3\lambda \llbracket X_{\varepsilon}^2
     \rrbracket \succ \phi_{\varepsilon}] \right\rangle_{\varepsilon} \]
  \[ + D_{\rho^4,\varepsilon} \left( \phi_{\varepsilon}, 3\lambda \llbracket X_{\varepsilon}^2
     \rrbracket, \Q_{\varepsilon}^{- 1} [3\lambda \llbracket X_{\varepsilon}^2
     \rrbracket \succ \phi_{\varepsilon}] \right), \]
  hence Lemma~\ref{lem:comm1} leads to
  \[ = \left\langle \rho^4 \phi_{\varepsilon}^2, 9\lambda^2 \llbracket
     X_{\varepsilon}^2 \rrbracket \circ \Q_{\varepsilon}^{- 1} \llbracket
     X_{\varepsilon}^2 \rrbracket \right\rangle_{\varepsilon} + \langle \rho^4
     \phi_{\varepsilon}, \tilde{C}_{\varepsilon} (\phi, 3\lambda \llbracket
     X_{\varepsilon}^2 \rrbracket, 3\lambda \llbracket X_{\varepsilon}^2 \rrbracket)
     \rangle_{\varepsilon} \]
  \[ + D_{\rho^4,\varepsilon} \left( \phi_{\varepsilon}, 3\lambda \llbracket X_{\varepsilon}^2
     \rrbracket, \Q_{\varepsilon}^{- 1} [3\lambda \llbracket X_{\varepsilon}^2
     \rrbracket \succ \phi_{\varepsilon}] \right) . \]
  We conclude that
  \[ \Phi_{\rho^4, \varepsilon} = - m^{2} \langle \rho^4 \psi_{\varepsilon},
     \psi_{\varepsilon} \rangle_{\varepsilon} - \langle \rho^4 \nabla_{\varepsilon}
     \psi_{\varepsilon}, \nabla_{\varepsilon} \psi_{\varepsilon} \rangle_{\varepsilon} 
     - \langle [\nabla_{\varepsilon}, \rho^4] \psi_{\varepsilon},
     \nabla_{\varepsilon} \psi_{\varepsilon} \rangle_{\varepsilon} \]
  \[  + \left\langle \left[
     \Q_{\varepsilon}, \rho^4 \right] \Q_{\varepsilon}^{- 1} [3\lambda \llbracket
     X_{\varepsilon}^2 \rrbracket \succ \phi_{\varepsilon}],
     \psi_{\varepsilon} \right\rangle_{\varepsilon} + \left\langle \rho^4
     \phi_{\varepsilon}^2, 9\lambda^2 \llbracket X_{\varepsilon}^2 \rrbracket \circ
     \Q_{\varepsilon}^{- 1} \llbracket X_{\varepsilon}^2 \rrbracket - 3 \lambda^2
     b_{\varepsilon} \right\rangle_{\varepsilon} \]
  \[ + D_{\rho^4, \varepsilon} (\phi_{\varepsilon}, - 3\lambda \llbracket
     X_{\varepsilon}^2 \rrbracket, \phi_{\varepsilon}) + \langle \rho^4
     \phi_{\varepsilon}, \tilde{C}_{\varepsilon} (\phi_{\varepsilon}, 3\lambda
     \llbracket X_{\varepsilon}^2 \rrbracket, 3\lambda \llbracket X_{\varepsilon}^2
     \rrbracket) \rangle_{\varepsilon} \]
  \[ + D_{\rho^4, \varepsilon} \left( \phi_{\varepsilon}, 3\lambda \llbracket
     X_{\varepsilon}^2 \rrbracket, \Q_{\varepsilon}^{- 1} [3\lambda \llbracket
     X_{\varepsilon}^2 \rrbracket \succ \phi_{\varepsilon}] \right) . \]
  As the next step, we justify the definition of the resonant product
  appearing in $\Psi_{\rho^4, \varepsilon}$ and show that it is given by
  $Z_{\varepsilon}$ from the statement of the lemma. To this end, let
  \[ Z_{\varepsilon} \assign - 3\lambda^{-1} \llbracket X_{\varepsilon}^2 \rrbracket \circ
     Y_{\varepsilon} - 3 b_{\varepsilon} (X_{\varepsilon} + Y_{\varepsilon}),
  \]
  and recall the definition of $Y_{M,\varepsilon}$ \eqref{eq:YY}. Hence by Lemma~\ref{lem:comm1}
  \[ Z_{\varepsilon} = 3 \llbracket X_{\varepsilon}^2 \rrbracket \circ
     X_{\varepsilon}^{\!\resizebox{0.6em}{!}{
\begin{tikzpicture}
\pgfpathmoveto{\pgfqpoint{0cm}{-0.035cm}}
\pgfpathlineto{\pgfqpoint{1.376cm}{-0.035cm}}
\pgfpathlineto{\pgfqpoint{1.376cm}{1.552cm}}
\pgfpathlineto{\pgfqpoint{0cm}{1.552cm}}
\pgfpathclose
\pgfusepath{clip}
\begin{pgfscope}
\begin{pgfscope}
\pgfpathmoveto{\pgfqpoint{0cm}{-0.035cm}}
\pgfpathlineto{\pgfqpoint{1.376cm}{-0.035cm}}
\pgfpathlineto{\pgfqpoint{1.376cm}{1.552cm}}
\pgfpathlineto{\pgfqpoint{0cm}{1.552cm}}
\pgfpathclose
\pgfusepath{clip}
\begin{pgfscope}
\begin{pgfscope}
\pgfsetdash{}{0cm}
\pgfsetlinewidth{0.818mm}
\pgfsetroundcap
\pgfsetroundjoin
\pgfsetmiterlimit{7.0}
\definecolor{eps2pgf_color}{gray}{0}\pgfsetstrokecolor{eps2pgf_color}\pgfsetfillcolor{eps2pgf_color}
\pgfpathmoveto{\pgfqpoint{0.117cm}{1.421cm}}
\pgfpathlineto{\pgfqpoint{0.682cm}{0.671cm}}
\pgfpathlineto{\pgfqpoint{1.246cm}{1.421cm}}
\pgfusepath{stroke}
\end{pgfscope}
\definecolor{eps2pgf_color}{gray}{0}\pgfsetstrokecolor{eps2pgf_color}\pgfsetfillcolor{eps2pgf_color}
\pgfpathmoveto{\pgfqpoint{0.273cm}{1.395cm}}
\pgfpathcurveto{\pgfqpoint{0.273cm}{1.432cm}}{\pgfqpoint{0.259cm}{1.467cm}}{\pgfqpoint{0.233cm}{1.492cm}}
\pgfpathcurveto{\pgfqpoint{0.207cm}{1.518cm}}{\pgfqpoint{0.173cm}{1.532cm}}{\pgfqpoint{0.137cm}{1.532cm}}
\pgfpathcurveto{\pgfqpoint{0.1cm}{1.532cm}}{\pgfqpoint{0.066cm}{1.518cm}}{\pgfqpoint{0.04cm}{1.492cm}}
\pgfpathcurveto{\pgfqpoint{0.014cm}{1.467cm}}{\pgfqpoint{0cm}{1.432cm}}{\pgfqpoint{0cm}{1.395cm}}
\pgfpathcurveto{\pgfqpoint{0cm}{1.359cm}}{\pgfqpoint{0.014cm}{1.324cm}}{\pgfqpoint{0.04cm}{1.299cm}}
\pgfpathcurveto{\pgfqpoint{0.066cm}{1.273cm}}{\pgfqpoint{0.1cm}{1.258cm}}{\pgfqpoint{0.137cm}{1.258cm}}
\pgfpathcurveto{\pgfqpoint{0.173cm}{1.258cm}}{\pgfqpoint{0.207cm}{1.273cm}}{\pgfqpoint{0.233cm}{1.299cm}}
\pgfpathcurveto{\pgfqpoint{0.259cm}{1.324cm}}{\pgfqpoint{0.273cm}{1.359cm}}{\pgfqpoint{0.273cm}{1.395cm}}
\pgfusepath{fill}
\begin{pgfscope}
\pgfsetdash{}{0cm}
\pgfsetlinewidth{0.818mm}
\pgfsetmiterlimit{7.0}
\pgfpathmoveto{\pgfqpoint{0.682cm}{0.671cm}}
\pgfpathlineto{\pgfqpoint{0.679cm}{1.418cm}}
\pgfusepath{stroke}
\end{pgfscope}
\pgfpathmoveto{\pgfqpoint{0.815cm}{1.399cm}}
\pgfpathcurveto{\pgfqpoint{0.815cm}{1.435cm}}{\pgfqpoint{0.801cm}{1.47cm}}{\pgfqpoint{0.775cm}{1.496cm}}
\pgfpathcurveto{\pgfqpoint{0.75cm}{1.521cm}}{\pgfqpoint{0.715cm}{1.536cm}}{\pgfqpoint{0.679cm}{1.536cm}}
\pgfpathcurveto{\pgfqpoint{0.643cm}{1.536cm}}{\pgfqpoint{0.608cm}{1.521cm}}{\pgfqpoint{0.582cm}{1.496cm}}
\pgfpathcurveto{\pgfqpoint{0.557cm}{1.47cm}}{\pgfqpoint{0.542cm}{1.435cm}}{\pgfqpoint{0.542cm}{1.399cm}}
\pgfpathcurveto{\pgfqpoint{0.542cm}{1.363cm}}{\pgfqpoint{0.557cm}{1.328cm}}{\pgfqpoint{0.582cm}{1.302cm}}
\pgfpathcurveto{\pgfqpoint{0.608cm}{1.276cm}}{\pgfqpoint{0.643cm}{1.262cm}}{\pgfqpoint{0.679cm}{1.262cm}}
\pgfpathcurveto{\pgfqpoint{0.715cm}{1.262cm}}{\pgfqpoint{0.75cm}{1.276cm}}{\pgfqpoint{0.775cm}{1.302cm}}
\pgfpathcurveto{\pgfqpoint{0.801cm}{1.328cm}}{\pgfqpoint{0.815cm}{1.363cm}}{\pgfqpoint{0.815cm}{1.399cm}}
\pgfusepath{fill}
\pgfpathmoveto{\pgfqpoint{1.345cm}{1.371cm}}
\pgfpathcurveto{\pgfqpoint{1.345cm}{1.408cm}}{\pgfqpoint{1.331cm}{1.442cm}}{\pgfqpoint{1.305cm}{1.468cm}}
\pgfpathcurveto{\pgfqpoint{1.28cm}{1.494cm}}{\pgfqpoint{1.245cm}{1.508cm}}{\pgfqpoint{1.209cm}{1.508cm}}
\pgfpathcurveto{\pgfqpoint{1.172cm}{1.508cm}}{\pgfqpoint{1.138cm}{1.494cm}}{\pgfqpoint{1.112cm}{1.468cm}}
\pgfpathcurveto{\pgfqpoint{1.087cm}{1.442cm}}{\pgfqpoint{1.072cm}{1.408cm}}{\pgfqpoint{1.072cm}{1.371cm}}
\pgfpathcurveto{\pgfqpoint{1.072cm}{1.335cm}}{\pgfqpoint{1.087cm}{1.3cm}}{\pgfqpoint{1.112cm}{1.274cm}}
\pgfpathcurveto{\pgfqpoint{1.138cm}{1.249cm}}{\pgfqpoint{1.172cm}{1.234cm}}{\pgfqpoint{1.209cm}{1.234cm}}
\pgfpathcurveto{\pgfqpoint{1.245cm}{1.234cm}}{\pgfqpoint{1.28cm}{1.249cm}}{\pgfqpoint{1.305cm}{1.274cm}}
\pgfpathcurveto{\pgfqpoint{1.331cm}{1.3cm}}{\pgfqpoint{1.345cm}{1.335cm}}{\pgfqpoint{1.345cm}{1.371cm}}
\pgfusepath{fill}
\begin{pgfscope}
\pgfsetdash{}{0cm}
\pgfsetlinewidth{0.818mm}
\pgfsetroundcap
\pgfsetmiterlimit{4.0}
\pgfpathmoveto{\pgfqpoint{0.682cm}{0.671cm}}
\pgfpathlineto{\pgfqpoint{0.682cm}{0.042cm}}
\pgfusepath{stroke}
\end{pgfscope}
\end{pgfscope}
\end{pgfscope}
\end{pgfscope}
\end{tikzpicture}}} - 3 b_{\varepsilon} X_{\varepsilon} + 3
     \llbracket X_{\varepsilon}^2 \rrbracket \circ \LL_{\varepsilon}^{- 1} (3
     \llbracket X_{\varepsilon}^2 \rrbracket \succ Y_{\varepsilon}) - 3
     b_{\varepsilon} Y_{\varepsilon} \]
  \[ - 3 \llbracket X_{\varepsilon}^2 \rrbracket \circ \LL_{\varepsilon}^{- 1}
     ( 3 \UU^{\varepsilon}_{\leqslant} \llbracket X_{\varepsilon}^2
     \rrbracket \succ Y_{\varepsilon} ) \]
  \[ = (3 \llbracket X_{\varepsilon}^2 \rrbracket \circ
     X_{\varepsilon}^{\!\resizebox{0.6em}{!}{
\begin{tikzpicture}
\pgfpathmoveto{\pgfqpoint{0cm}{-0.035cm}}
\pgfpathlineto{\pgfqpoint{1.376cm}{-0.035cm}}
\pgfpathlineto{\pgfqpoint{1.376cm}{1.552cm}}
\pgfpathlineto{\pgfqpoint{0cm}{1.552cm}}
\pgfpathclose
\pgfusepath{clip}
\begin{pgfscope}
\begin{pgfscope}
\pgfpathmoveto{\pgfqpoint{0cm}{-0.035cm}}
\pgfpathlineto{\pgfqpoint{1.376cm}{-0.035cm}}
\pgfpathlineto{\pgfqpoint{1.376cm}{1.552cm}}
\pgfpathlineto{\pgfqpoint{0cm}{1.552cm}}
\pgfpathclose
\pgfusepath{clip}
\begin{pgfscope}
\begin{pgfscope}
\pgfsetdash{}{0cm}
\pgfsetlinewidth{0.818mm}
\pgfsetroundcap
\pgfsetroundjoin
\pgfsetmiterlimit{7.0}
\definecolor{eps2pgf_color}{gray}{0}\pgfsetstrokecolor{eps2pgf_color}\pgfsetfillcolor{eps2pgf_color}
\pgfpathmoveto{\pgfqpoint{0.117cm}{1.421cm}}
\pgfpathlineto{\pgfqpoint{0.682cm}{0.671cm}}
\pgfpathlineto{\pgfqpoint{1.246cm}{1.421cm}}
\pgfusepath{stroke}
\end{pgfscope}
\definecolor{eps2pgf_color}{gray}{0}\pgfsetstrokecolor{eps2pgf_color}\pgfsetfillcolor{eps2pgf_color}
\pgfpathmoveto{\pgfqpoint{0.273cm}{1.395cm}}
\pgfpathcurveto{\pgfqpoint{0.273cm}{1.432cm}}{\pgfqpoint{0.259cm}{1.467cm}}{\pgfqpoint{0.233cm}{1.492cm}}
\pgfpathcurveto{\pgfqpoint{0.207cm}{1.518cm}}{\pgfqpoint{0.173cm}{1.532cm}}{\pgfqpoint{0.137cm}{1.532cm}}
\pgfpathcurveto{\pgfqpoint{0.1cm}{1.532cm}}{\pgfqpoint{0.066cm}{1.518cm}}{\pgfqpoint{0.04cm}{1.492cm}}
\pgfpathcurveto{\pgfqpoint{0.014cm}{1.467cm}}{\pgfqpoint{0cm}{1.432cm}}{\pgfqpoint{0cm}{1.395cm}}
\pgfpathcurveto{\pgfqpoint{0cm}{1.359cm}}{\pgfqpoint{0.014cm}{1.324cm}}{\pgfqpoint{0.04cm}{1.299cm}}
\pgfpathcurveto{\pgfqpoint{0.066cm}{1.273cm}}{\pgfqpoint{0.1cm}{1.258cm}}{\pgfqpoint{0.137cm}{1.258cm}}
\pgfpathcurveto{\pgfqpoint{0.173cm}{1.258cm}}{\pgfqpoint{0.207cm}{1.273cm}}{\pgfqpoint{0.233cm}{1.299cm}}
\pgfpathcurveto{\pgfqpoint{0.259cm}{1.324cm}}{\pgfqpoint{0.273cm}{1.359cm}}{\pgfqpoint{0.273cm}{1.395cm}}
\pgfusepath{fill}
\begin{pgfscope}
\pgfsetdash{}{0cm}
\pgfsetlinewidth{0.818mm}
\pgfsetmiterlimit{7.0}
\pgfpathmoveto{\pgfqpoint{0.682cm}{0.671cm}}
\pgfpathlineto{\pgfqpoint{0.679cm}{1.418cm}}
\pgfusepath{stroke}
\end{pgfscope}
\pgfpathmoveto{\pgfqpoint{0.815cm}{1.399cm}}
\pgfpathcurveto{\pgfqpoint{0.815cm}{1.435cm}}{\pgfqpoint{0.801cm}{1.47cm}}{\pgfqpoint{0.775cm}{1.496cm}}
\pgfpathcurveto{\pgfqpoint{0.75cm}{1.521cm}}{\pgfqpoint{0.715cm}{1.536cm}}{\pgfqpoint{0.679cm}{1.536cm}}
\pgfpathcurveto{\pgfqpoint{0.643cm}{1.536cm}}{\pgfqpoint{0.608cm}{1.521cm}}{\pgfqpoint{0.582cm}{1.496cm}}
\pgfpathcurveto{\pgfqpoint{0.557cm}{1.47cm}}{\pgfqpoint{0.542cm}{1.435cm}}{\pgfqpoint{0.542cm}{1.399cm}}
\pgfpathcurveto{\pgfqpoint{0.542cm}{1.363cm}}{\pgfqpoint{0.557cm}{1.328cm}}{\pgfqpoint{0.582cm}{1.302cm}}
\pgfpathcurveto{\pgfqpoint{0.608cm}{1.276cm}}{\pgfqpoint{0.643cm}{1.262cm}}{\pgfqpoint{0.679cm}{1.262cm}}
\pgfpathcurveto{\pgfqpoint{0.715cm}{1.262cm}}{\pgfqpoint{0.75cm}{1.276cm}}{\pgfqpoint{0.775cm}{1.302cm}}
\pgfpathcurveto{\pgfqpoint{0.801cm}{1.328cm}}{\pgfqpoint{0.815cm}{1.363cm}}{\pgfqpoint{0.815cm}{1.399cm}}
\pgfusepath{fill}
\pgfpathmoveto{\pgfqpoint{1.345cm}{1.371cm}}
\pgfpathcurveto{\pgfqpoint{1.345cm}{1.408cm}}{\pgfqpoint{1.331cm}{1.442cm}}{\pgfqpoint{1.305cm}{1.468cm}}
\pgfpathcurveto{\pgfqpoint{1.28cm}{1.494cm}}{\pgfqpoint{1.245cm}{1.508cm}}{\pgfqpoint{1.209cm}{1.508cm}}
\pgfpathcurveto{\pgfqpoint{1.172cm}{1.508cm}}{\pgfqpoint{1.138cm}{1.494cm}}{\pgfqpoint{1.112cm}{1.468cm}}
\pgfpathcurveto{\pgfqpoint{1.087cm}{1.442cm}}{\pgfqpoint{1.072cm}{1.408cm}}{\pgfqpoint{1.072cm}{1.371cm}}
\pgfpathcurveto{\pgfqpoint{1.072cm}{1.335cm}}{\pgfqpoint{1.087cm}{1.3cm}}{\pgfqpoint{1.112cm}{1.274cm}}
\pgfpathcurveto{\pgfqpoint{1.138cm}{1.249cm}}{\pgfqpoint{1.172cm}{1.234cm}}{\pgfqpoint{1.209cm}{1.234cm}}
\pgfpathcurveto{\pgfqpoint{1.245cm}{1.234cm}}{\pgfqpoint{1.28cm}{1.249cm}}{\pgfqpoint{1.305cm}{1.274cm}}
\pgfpathcurveto{\pgfqpoint{1.331cm}{1.3cm}}{\pgfqpoint{1.345cm}{1.335cm}}{\pgfqpoint{1.345cm}{1.371cm}}
\pgfusepath{fill}
\begin{pgfscope}
\pgfsetdash{}{0cm}
\pgfsetlinewidth{0.818mm}
\pgfsetroundcap
\pgfsetmiterlimit{4.0}
\pgfpathmoveto{\pgfqpoint{0.682cm}{0.671cm}}
\pgfpathlineto{\pgfqpoint{0.682cm}{0.042cm}}
\pgfusepath{stroke}
\end{pgfscope}
\end{pgfscope}
\end{pgfscope}
\end{pgfscope}
\end{tikzpicture}}} - 3 b_{\varepsilon} X_{\varepsilon}) + (
     3 \llbracket X_{\varepsilon}^2 \rrbracket \circ \LL_{\varepsilon}^{- 1} 3
     \llbracket X_{\varepsilon}^2 \rrbracket - 3 \tilde{b}_{\varepsilon}
     ) Y_{\varepsilon} + 3 (\tilde{b}_{\varepsilon} - b_{\varepsilon})
     Y_{\varepsilon} \]
  \[ + \bar{C}_{\varepsilon} (Y_{\varepsilon}, 3 \llbracket X_{\varepsilon}^2
     \rrbracket, 3 \llbracket X_{\varepsilon}^2 \rrbracket) - 3 \llbracket
     X_{\varepsilon}^2 \rrbracket \circ \LL_{\varepsilon}^{- 1} \left( 3
     \UU_{\leqslant} \llbracket X_{\varepsilon}^2 \rrbracket \succ
     Y_{\varepsilon} \right), \]
  which is the desired formula. In this formulation we clearly see
  the structure of the renormalization and the appropriate combinations of
  resonant products and the counterterms.
\end{proof}

As the next step, we estimate the new stochastic terms appearing in Lemma
\ref{lem:energy12}. Here and in the sequel, $\vartheta=O(\kappa)>0$ denotes a generic small constant which changes from line to line.

\begin{lemma}
  \label{lem:Z}
It holds true
  \[ \| Z_{\varepsilon} (t) \|_{\CC^{- 1 / 2 - \kappa, \varepsilon}
     (\rho^{\sigma})} \lesssim (1+\lambda |\log t| +\lambda^2)  \|     \mathbb{X}_{\varepsilon}\|^{7+\vartheta}, \]
  \[ \| X_{\varepsilon} Y_{\varepsilon} \|_{C_T \CC^{- 1 / 2 - \kappa,
     \varepsilon} (\rho^{\sigma})}  \lesssim (\lambda+\lambda^2)
    \|     \mathbb{X}_{\varepsilon}\|^{6}, \]
     \[ \| X_{\varepsilon} Y_{\varepsilon}^2
     \|_{C_T \CC^{- 1 / 2 - \kappa, \varepsilon} (\rho^{\sigma})} \lesssim (\lambda^{2}+\lambda^3)
     \|     \mathbb{X}_{\varepsilon}\|^{9}. \]
\end{lemma}

\begin{proof}
  By definition of $Z_{\varepsilon}$ and the discussion in Section~\ref{ssec:stoch}, Lemma~\ref{lem:Y1}, Lemma~\ref{lem:comm1}, Lemma~\ref{lem:loc} and {\eqref{eq:U11}} we have (since the choice of exponent
  $\sigma > 0$ of the weight corresponding to the stochastic objects is
  arbitrary, $\sigma$ changes from line to line in the sequel)
  \[ \| Z_{\varepsilon} (t) \|_{\CC^{- 1 / 2 - \kappa, \varepsilon} (\rho^{3
     \sigma})} \lesssim \| X_{\varepsilon}^{\!\resizebox{!}{.8em}{
\begin{tikzpicture}
\pgfpathmoveto{\pgfqpoint{0cm}{-0.035cm}}
\pgfpathlineto{\pgfqpoint{1.976cm}{-0.035cm}}
\pgfpathlineto{\pgfqpoint{1.976cm}{1.94cm}}
\pgfpathlineto{\pgfqpoint{0cm}{1.94cm}}
\pgfpathclose
\pgfusepath{clip}
\begin{pgfscope}
\begin{pgfscope}
\pgfpathmoveto{\pgfqpoint{0cm}{-0.035cm}}
\pgfpathlineto{\pgfqpoint{1.976cm}{-0.035cm}}
\pgfpathlineto{\pgfqpoint{1.976cm}{1.94cm}}
\pgfpathlineto{\pgfqpoint{0cm}{1.94cm}}
\pgfpathclose
\pgfusepath{clip}
\begin{pgfscope}
\begin{pgfscope}
\pgfsetdash{}{0cm}
\pgfsetlinewidth{0.818mm}
\pgfsetroundcap
\pgfsetroundjoin
\pgfsetmiterlimit{7.0}
\definecolor{eps2pgf_color}{gray}{0}\pgfsetstrokecolor{eps2pgf_color}\pgfsetfillcolor{eps2pgf_color}
\pgfpathmoveto{\pgfqpoint{0.117cm}{1.815cm}}
\pgfpathlineto{\pgfqpoint{0.682cm}{1.065cm}}
\pgfpathlineto{\pgfqpoint{1.246cm}{1.815cm}}
\pgfusepath{stroke}
\end{pgfscope}
\definecolor{eps2pgf_color}{gray}{0}\pgfsetstrokecolor{eps2pgf_color}\pgfsetfillcolor{eps2pgf_color}
\pgfpathmoveto{\pgfqpoint{0.273cm}{1.789cm}}
\pgfpathcurveto{\pgfqpoint{0.273cm}{1.825cm}}{\pgfqpoint{0.259cm}{1.86cm}}{\pgfqpoint{0.233cm}{1.886cm}}
\pgfpathcurveto{\pgfqpoint{0.207cm}{1.912cm}}{\pgfqpoint{0.173cm}{1.926cm}}{\pgfqpoint{0.137cm}{1.926cm}}
\pgfpathcurveto{\pgfqpoint{0.1cm}{1.926cm}}{\pgfqpoint{0.066cm}{1.912cm}}{\pgfqpoint{0.04cm}{1.886cm}}
\pgfpathcurveto{\pgfqpoint{0.014cm}{1.86cm}}{\pgfqpoint{0cm}{1.825cm}}{\pgfqpoint{0cm}{1.789cm}}
\pgfpathcurveto{\pgfqpoint{0cm}{1.753cm}}{\pgfqpoint{0.014cm}{1.718cm}}{\pgfqpoint{0.04cm}{1.692cm}}
\pgfpathcurveto{\pgfqpoint{0.066cm}{1.667cm}}{\pgfqpoint{0.1cm}{1.652cm}}{\pgfqpoint{0.137cm}{1.652cm}}
\pgfpathcurveto{\pgfqpoint{0.173cm}{1.652cm}}{\pgfqpoint{0.207cm}{1.667cm}}{\pgfqpoint{0.233cm}{1.692cm}}
\pgfpathcurveto{\pgfqpoint{0.259cm}{1.718cm}}{\pgfqpoint{0.273cm}{1.753cm}}{\pgfqpoint{0.273cm}{1.789cm}}
\pgfusepath{fill}
\begin{pgfscope}
\pgfsetdash{}{0cm}
\pgfsetlinewidth{0.818mm}
\pgfsetmiterlimit{7.0}
\pgfpathmoveto{\pgfqpoint{0.682cm}{1.065cm}}
\pgfpathlineto{\pgfqpoint{0.679cm}{1.812cm}}
\pgfusepath{stroke}
\end{pgfscope}
\pgfpathmoveto{\pgfqpoint{0.815cm}{1.793cm}}
\pgfpathcurveto{\pgfqpoint{0.815cm}{1.829cm}}{\pgfqpoint{0.801cm}{1.864cm}}{\pgfqpoint{0.775cm}{1.89cm}}
\pgfpathcurveto{\pgfqpoint{0.75cm}{1.915cm}}{\pgfqpoint{0.715cm}{1.93cm}}{\pgfqpoint{0.679cm}{1.93cm}}
\pgfpathcurveto{\pgfqpoint{0.643cm}{1.93cm}}{\pgfqpoint{0.608cm}{1.915cm}}{\pgfqpoint{0.582cm}{1.89cm}}
\pgfpathcurveto{\pgfqpoint{0.557cm}{1.864cm}}{\pgfqpoint{0.542cm}{1.829cm}}{\pgfqpoint{0.542cm}{1.793cm}}
\pgfpathcurveto{\pgfqpoint{0.542cm}{1.756cm}}{\pgfqpoint{0.557cm}{1.722cm}}{\pgfqpoint{0.582cm}{1.696cm}}
\pgfpathcurveto{\pgfqpoint{0.608cm}{1.67cm}}{\pgfqpoint{0.643cm}{1.656cm}}{\pgfqpoint{0.679cm}{1.656cm}}
\pgfpathcurveto{\pgfqpoint{0.715cm}{1.656cm}}{\pgfqpoint{0.75cm}{1.67cm}}{\pgfqpoint{0.775cm}{1.696cm}}
\pgfpathcurveto{\pgfqpoint{0.801cm}{1.722cm}}{\pgfqpoint{0.815cm}{1.756cm}}{\pgfqpoint{0.815cm}{1.793cm}}
\pgfusepath{fill}
\pgfpathmoveto{\pgfqpoint{1.345cm}{1.765cm}}
\pgfpathcurveto{\pgfqpoint{1.345cm}{1.801cm}}{\pgfqpoint{1.331cm}{1.836cm}}{\pgfqpoint{1.305cm}{1.862cm}}
\pgfpathcurveto{\pgfqpoint{1.28cm}{1.887cm}}{\pgfqpoint{1.245cm}{1.902cm}}{\pgfqpoint{1.209cm}{1.902cm}}
\pgfpathcurveto{\pgfqpoint{1.172cm}{1.902cm}}{\pgfqpoint{1.138cm}{1.887cm}}{\pgfqpoint{1.112cm}{1.862cm}}
\pgfpathcurveto{\pgfqpoint{1.087cm}{1.836cm}}{\pgfqpoint{1.072cm}{1.801cm}}{\pgfqpoint{1.072cm}{1.765cm}}
\pgfpathcurveto{\pgfqpoint{1.072cm}{1.728cm}}{\pgfqpoint{1.087cm}{1.694cm}}{\pgfqpoint{1.112cm}{1.668cm}}
\pgfpathcurveto{\pgfqpoint{1.138cm}{1.642cm}}{\pgfqpoint{1.172cm}{1.628cm}}{\pgfqpoint{1.209cm}{1.628cm}}
\pgfpathcurveto{\pgfqpoint{1.245cm}{1.628cm}}{\pgfqpoint{1.28cm}{1.642cm}}{\pgfqpoint{1.305cm}{1.668cm}}
\pgfpathcurveto{\pgfqpoint{1.331cm}{1.694cm}}{\pgfqpoint{1.345cm}{1.728cm}}{\pgfqpoint{1.345cm}{1.765cm}}
\pgfusepath{fill}
\begin{pgfscope}
\pgfsetdash{}{0cm}
\pgfsetlinewidth{0.818mm}
\pgfsetroundcap
\pgfsetroundjoin
\pgfsetmiterlimit{7.0}
\pgfpathmoveto{\pgfqpoint{0.682cm}{1.065cm}}
\pgfpathlineto{\pgfqpoint{1.246cm}{0.315cm}}
\pgfpathlineto{\pgfqpoint{1.811cm}{1.065cm}}
\pgfusepath{stroke}
\end{pgfscope}
\pgfpathmoveto{\pgfqpoint{1.948cm}{1.065cm}}
\pgfpathcurveto{\pgfqpoint{1.948cm}{1.101cm}}{\pgfqpoint{1.933cm}{1.136cm}}{\pgfqpoint{1.907cm}{1.162cm}}
\pgfpathcurveto{\pgfqpoint{1.882cm}{1.187cm}}{\pgfqpoint{1.847cm}{1.202cm}}{\pgfqpoint{1.811cm}{1.202cm}}
\pgfpathcurveto{\pgfqpoint{1.775cm}{1.202cm}}{\pgfqpoint{1.74cm}{1.187cm}}{\pgfqpoint{1.714cm}{1.162cm}}
\pgfpathcurveto{\pgfqpoint{1.689cm}{1.136cm}}{\pgfqpoint{1.674cm}{1.101cm}}{\pgfqpoint{1.674cm}{1.065cm}}
\pgfpathcurveto{\pgfqpoint{1.674cm}{1.029cm}}{\pgfqpoint{1.689cm}{0.994cm}}{\pgfqpoint{1.714cm}{0.968cm}}
\pgfpathcurveto{\pgfqpoint{1.74cm}{0.942cm}}{\pgfqpoint{1.775cm}{0.928cm}}{\pgfqpoint{1.811cm}{0.928cm}}
\pgfpathcurveto{\pgfqpoint{1.847cm}{0.928cm}}{\pgfqpoint{1.882cm}{0.942cm}}{\pgfqpoint{1.907cm}{0.968cm}}
\pgfpathcurveto{\pgfqpoint{1.933cm}{0.994cm}}{\pgfqpoint{1.948cm}{1.029cm}}{\pgfqpoint{1.948cm}{1.065cm}}
\pgfusepath{fill}
\begin{pgfscope}
\pgfsetdash{}{0cm}
\pgfsetlinewidth{0.818mm}
\pgfsetmiterlimit{7.0}
\pgfpathmoveto{\pgfqpoint{1.246cm}{0.315cm}}
\pgfpathlineto{\pgfqpoint{1.244cm}{1.061cm}}
\pgfusepath{stroke}
\end{pgfscope}
\pgfpathmoveto{\pgfqpoint{1.38cm}{1.065cm}}
\pgfpathcurveto{\pgfqpoint{1.38cm}{1.101cm}}{\pgfqpoint{1.366cm}{1.136cm}}{\pgfqpoint{1.34cm}{1.162cm}}
\pgfpathcurveto{\pgfqpoint{1.315cm}{1.187cm}}{\pgfqpoint{1.28cm}{1.202cm}}{\pgfqpoint{1.244cm}{1.202cm}}
\pgfpathcurveto{\pgfqpoint{1.207cm}{1.202cm}}{\pgfqpoint{1.173cm}{1.187cm}}{\pgfqpoint{1.147cm}{1.162cm}}
\pgfpathcurveto{\pgfqpoint{1.121cm}{1.136cm}}{\pgfqpoint{1.107cm}{1.101cm}}{\pgfqpoint{1.107cm}{1.065cm}}
\pgfpathcurveto{\pgfqpoint{1.107cm}{1.029cm}}{\pgfqpoint{1.121cm}{0.994cm}}{\pgfqpoint{1.147cm}{0.968cm}}
\pgfpathcurveto{\pgfqpoint{1.173cm}{0.942cm}}{\pgfqpoint{1.207cm}{0.928cm}}{\pgfqpoint{1.244cm}{0.928cm}}
\pgfpathcurveto{\pgfqpoint{1.28cm}{0.928cm}}{\pgfqpoint{1.315cm}{0.942cm}}{\pgfqpoint{1.34cm}{0.968cm}}
\pgfpathcurveto{\pgfqpoint{1.366cm}{0.994cm}}{\pgfqpoint{1.38cm}{1.029cm}}{\pgfqpoint{1.38cm}{1.065cm}}
\pgfusepath{fill}
\begin{pgfscope}
\pgfsetdash{}{0cm}
\pgfsetlinewidth{0.818mm}
\pgfsetmiterlimit{4.0}
\pgfpathmoveto{\pgfqpoint{1.383cm}{0.178cm}}
\pgfpathcurveto{\pgfqpoint{1.383cm}{0.214cm}}{\pgfqpoint{1.369cm}{0.249cm}}{\pgfqpoint{1.343cm}{0.275cm}}
\pgfpathcurveto{\pgfqpoint{1.317cm}{0.3cm}}{\pgfqpoint{1.283cm}{0.315cm}}{\pgfqpoint{1.246cm}{0.315cm}}
\pgfpathcurveto{\pgfqpoint{1.21cm}{0.315cm}}{\pgfqpoint{1.175cm}{0.3cm}}{\pgfqpoint{1.15cm}{0.275cm}}
\pgfpathcurveto{\pgfqpoint{1.124cm}{0.249cm}}{\pgfqpoint{1.11cm}{0.214cm}}{\pgfqpoint{1.11cm}{0.178cm}}
\pgfpathcurveto{\pgfqpoint{1.11cm}{0.141cm}}{\pgfqpoint{1.124cm}{0.107cm}}{\pgfqpoint{1.15cm}{0.081cm}}
\pgfpathcurveto{\pgfqpoint{1.175cm}{0.055cm}}{\pgfqpoint{1.21cm}{0.041cm}}{\pgfqpoint{1.246cm}{0.041cm}}
\pgfpathcurveto{\pgfqpoint{1.283cm}{0.041cm}}{\pgfqpoint{1.317cm}{0.055cm}}{\pgfqpoint{1.343cm}{0.081cm}}
\pgfpathcurveto{\pgfqpoint{1.369cm}{0.107cm}}{\pgfqpoint{1.383cm}{0.141cm}}{\pgfqpoint{1.383cm}{0.178cm}}
\pgfusepath{stroke}
\end{pgfscope}
\end{pgfscope}
\end{pgfscope}
\end{pgfscope}
\end{tikzpicture}}} \|_{C_T
     \CC^{- 1 / 2 - \kappa, \varepsilon} (\rho^{3 \sigma})} + \|
     \tilde{X}_{\varepsilon}^{\!\resizebox{!}{.8em}{
\begin{tikzpicture}
\pgfpathmoveto{\pgfqpoint{0cm}{-0.035cm}}
\pgfpathlineto{\pgfqpoint{1.976cm}{-0.035cm}}
\pgfpathlineto{\pgfqpoint{1.976cm}{1.94cm}}
\pgfpathlineto{\pgfqpoint{0cm}{1.94cm}}
\pgfpathclose
\pgfusepath{clip}
\begin{pgfscope}
\begin{pgfscope}
\pgfpathmoveto{\pgfqpoint{0cm}{-0.035cm}}
\pgfpathlineto{\pgfqpoint{1.976cm}{-0.035cm}}
\pgfpathlineto{\pgfqpoint{1.976cm}{1.94cm}}
\pgfpathlineto{\pgfqpoint{0cm}{1.94cm}}
\pgfpathclose
\pgfusepath{clip}
\begin{pgfscope}
\begin{pgfscope}
\pgfsetdash{}{0cm}
\pgfsetlinewidth{0.818mm}
\pgfsetroundcap
\pgfsetroundjoin
\pgfsetmiterlimit{7.0}
\definecolor{eps2pgf_color}{gray}{0}\pgfsetstrokecolor{eps2pgf_color}\pgfsetfillcolor{eps2pgf_color}
\pgfpathmoveto{\pgfqpoint{0.117cm}{1.815cm}}
\pgfpathlineto{\pgfqpoint{0.682cm}{1.065cm}}
\pgfpathlineto{\pgfqpoint{1.246cm}{1.815cm}}
\pgfusepath{stroke}
\end{pgfscope}
\definecolor{eps2pgf_color}{gray}{0}\pgfsetstrokecolor{eps2pgf_color}\pgfsetfillcolor{eps2pgf_color}
\pgfpathmoveto{\pgfqpoint{0.273cm}{1.789cm}}
\pgfpathcurveto{\pgfqpoint{0.273cm}{1.825cm}}{\pgfqpoint{0.259cm}{1.86cm}}{\pgfqpoint{0.233cm}{1.886cm}}
\pgfpathcurveto{\pgfqpoint{0.207cm}{1.912cm}}{\pgfqpoint{0.173cm}{1.926cm}}{\pgfqpoint{0.137cm}{1.926cm}}
\pgfpathcurveto{\pgfqpoint{0.1cm}{1.926cm}}{\pgfqpoint{0.066cm}{1.912cm}}{\pgfqpoint{0.04cm}{1.886cm}}
\pgfpathcurveto{\pgfqpoint{0.014cm}{1.86cm}}{\pgfqpoint{0cm}{1.825cm}}{\pgfqpoint{0cm}{1.789cm}}
\pgfpathcurveto{\pgfqpoint{0cm}{1.753cm}}{\pgfqpoint{0.014cm}{1.718cm}}{\pgfqpoint{0.04cm}{1.692cm}}
\pgfpathcurveto{\pgfqpoint{0.066cm}{1.667cm}}{\pgfqpoint{0.1cm}{1.652cm}}{\pgfqpoint{0.137cm}{1.652cm}}
\pgfpathcurveto{\pgfqpoint{0.173cm}{1.652cm}}{\pgfqpoint{0.207cm}{1.667cm}}{\pgfqpoint{0.233cm}{1.692cm}}
\pgfpathcurveto{\pgfqpoint{0.259cm}{1.718cm}}{\pgfqpoint{0.273cm}{1.753cm}}{\pgfqpoint{0.273cm}{1.789cm}}
\pgfusepath{fill}
\pgfpathmoveto{\pgfqpoint{1.345cm}{1.765cm}}
\pgfpathcurveto{\pgfqpoint{1.345cm}{1.801cm}}{\pgfqpoint{1.331cm}{1.836cm}}{\pgfqpoint{1.305cm}{1.862cm}}
\pgfpathcurveto{\pgfqpoint{1.28cm}{1.887cm}}{\pgfqpoint{1.245cm}{1.902cm}}{\pgfqpoint{1.209cm}{1.902cm}}
\pgfpathcurveto{\pgfqpoint{1.172cm}{1.902cm}}{\pgfqpoint{1.138cm}{1.887cm}}{\pgfqpoint{1.112cm}{1.862cm}}
\pgfpathcurveto{\pgfqpoint{1.087cm}{1.836cm}}{\pgfqpoint{1.072cm}{1.801cm}}{\pgfqpoint{1.072cm}{1.765cm}}
\pgfpathcurveto{\pgfqpoint{1.072cm}{1.728cm}}{\pgfqpoint{1.087cm}{1.694cm}}{\pgfqpoint{1.112cm}{1.668cm}}
\pgfpathcurveto{\pgfqpoint{1.138cm}{1.642cm}}{\pgfqpoint{1.172cm}{1.628cm}}{\pgfqpoint{1.209cm}{1.628cm}}
\pgfpathcurveto{\pgfqpoint{1.245cm}{1.628cm}}{\pgfqpoint{1.28cm}{1.642cm}}{\pgfqpoint{1.305cm}{1.668cm}}
\pgfpathcurveto{\pgfqpoint{1.331cm}{1.694cm}}{\pgfqpoint{1.345cm}{1.728cm}}{\pgfqpoint{1.345cm}{1.765cm}}
\pgfusepath{fill}
\begin{pgfscope}
\pgfsetdash{}{0cm}
\pgfsetlinewidth{0.818mm}
\pgfsetroundcap
\pgfsetroundjoin
\pgfsetmiterlimit{7.0}
\pgfpathmoveto{\pgfqpoint{0.682cm}{1.065cm}}
\pgfpathlineto{\pgfqpoint{1.246cm}{0.315cm}}
\pgfpathlineto{\pgfqpoint{1.811cm}{1.065cm}}
\pgfusepath{stroke}
\end{pgfscope}
\pgfpathmoveto{\pgfqpoint{1.948cm}{1.065cm}}
\pgfpathcurveto{\pgfqpoint{1.948cm}{1.101cm}}{\pgfqpoint{1.933cm}{1.136cm}}{\pgfqpoint{1.907cm}{1.162cm}}
\pgfpathcurveto{\pgfqpoint{1.882cm}{1.187cm}}{\pgfqpoint{1.847cm}{1.202cm}}{\pgfqpoint{1.811cm}{1.202cm}}
\pgfpathcurveto{\pgfqpoint{1.775cm}{1.202cm}}{\pgfqpoint{1.74cm}{1.187cm}}{\pgfqpoint{1.714cm}{1.162cm}}
\pgfpathcurveto{\pgfqpoint{1.689cm}{1.136cm}}{\pgfqpoint{1.674cm}{1.101cm}}{\pgfqpoint{1.674cm}{1.065cm}}
\pgfpathcurveto{\pgfqpoint{1.674cm}{1.029cm}}{\pgfqpoint{1.689cm}{0.994cm}}{\pgfqpoint{1.714cm}{0.968cm}}
\pgfpathcurveto{\pgfqpoint{1.74cm}{0.942cm}}{\pgfqpoint{1.775cm}{0.928cm}}{\pgfqpoint{1.811cm}{0.928cm}}
\pgfpathcurveto{\pgfqpoint{1.847cm}{0.928cm}}{\pgfqpoint{1.882cm}{0.942cm}}{\pgfqpoint{1.907cm}{0.968cm}}
\pgfpathcurveto{\pgfqpoint{1.933cm}{0.994cm}}{\pgfqpoint{1.948cm}{1.029cm}}{\pgfqpoint{1.948cm}{1.065cm}}
\pgfusepath{fill}
\begin{pgfscope}
\pgfsetdash{}{0cm}
\pgfsetlinewidth{0.818mm}
\pgfsetmiterlimit{7.0}
\pgfpathmoveto{\pgfqpoint{1.246cm}{0.315cm}}
\pgfpathlineto{\pgfqpoint{1.244cm}{1.061cm}}
\pgfusepath{stroke}
\end{pgfscope}
\pgfpathmoveto{\pgfqpoint{1.38cm}{1.065cm}}
\pgfpathcurveto{\pgfqpoint{1.38cm}{1.101cm}}{\pgfqpoint{1.366cm}{1.136cm}}{\pgfqpoint{1.34cm}{1.162cm}}
\pgfpathcurveto{\pgfqpoint{1.315cm}{1.187cm}}{\pgfqpoint{1.28cm}{1.202cm}}{\pgfqpoint{1.244cm}{1.202cm}}
\pgfpathcurveto{\pgfqpoint{1.207cm}{1.202cm}}{\pgfqpoint{1.173cm}{1.187cm}}{\pgfqpoint{1.147cm}{1.162cm}}
\pgfpathcurveto{\pgfqpoint{1.121cm}{1.136cm}}{\pgfqpoint{1.107cm}{1.101cm}}{\pgfqpoint{1.107cm}{1.065cm}}
\pgfpathcurveto{\pgfqpoint{1.107cm}{1.029cm}}{\pgfqpoint{1.121cm}{0.994cm}}{\pgfqpoint{1.147cm}{0.968cm}}
\pgfpathcurveto{\pgfqpoint{1.173cm}{0.942cm}}{\pgfqpoint{1.207cm}{0.928cm}}{\pgfqpoint{1.244cm}{0.928cm}}
\pgfpathcurveto{\pgfqpoint{1.28cm}{0.928cm}}{\pgfqpoint{1.315cm}{0.942cm}}{\pgfqpoint{1.34cm}{0.968cm}}
\pgfpathcurveto{\pgfqpoint{1.366cm}{0.994cm}}{\pgfqpoint{1.38cm}{1.029cm}}{\pgfqpoint{1.38cm}{1.065cm}}
\pgfusepath{fill}
\begin{pgfscope}
\pgfsetdash{}{0cm}
\pgfsetlinewidth{0.818mm}
\pgfsetmiterlimit{4.0}
\pgfpathmoveto{\pgfqpoint{1.383cm}{0.178cm}}
\pgfpathcurveto{\pgfqpoint{1.383cm}{0.214cm}}{\pgfqpoint{1.369cm}{0.249cm}}{\pgfqpoint{1.343cm}{0.275cm}}
\pgfpathcurveto{\pgfqpoint{1.317cm}{0.3cm}}{\pgfqpoint{1.283cm}{0.315cm}}{\pgfqpoint{1.246cm}{0.315cm}}
\pgfpathcurveto{\pgfqpoint{1.21cm}{0.315cm}}{\pgfqpoint{1.175cm}{0.3cm}}{\pgfqpoint{1.15cm}{0.275cm}}
\pgfpathcurveto{\pgfqpoint{1.124cm}{0.249cm}}{\pgfqpoint{1.11cm}{0.214cm}}{\pgfqpoint{1.11cm}{0.178cm}}
\pgfpathcurveto{\pgfqpoint{1.11cm}{0.141cm}}{\pgfqpoint{1.124cm}{0.107cm}}{\pgfqpoint{1.15cm}{0.081cm}}
\pgfpathcurveto{\pgfqpoint{1.175cm}{0.055cm}}{\pgfqpoint{1.21cm}{0.041cm}}{\pgfqpoint{1.246cm}{0.041cm}}
\pgfpathcurveto{\pgfqpoint{1.283cm}{0.041cm}}{\pgfqpoint{1.317cm}{0.055cm}}{\pgfqpoint{1.343cm}{0.081cm}}
\pgfpathcurveto{\pgfqpoint{1.369cm}{0.107cm}}{\pgfqpoint{1.383cm}{0.141cm}}{\pgfqpoint{1.383cm}{0.178cm}}
\pgfusepath{stroke}
\end{pgfscope}
\end{pgfscope}
\end{pgfscope}
\end{pgfscope}
\end{tikzpicture}}} \|_{C_T \CC^{- \kappa,
     \varepsilon} (\rho^{\sigma})} \| Y_{\varepsilon} \|_{C_T \CC^{1 / 2 -
     \kappa, \varepsilon} (\rho^{\sigma})} \]
  \[ + | \log t | \| Y_{\varepsilon} \|_{C_T \CC^{1 / 2 - \kappa, \varepsilon}
     (\rho^{\sigma})} + \big( \| Y_{\varepsilon} \|_{C \CC^{1 / 2 - \kappa,
     \varepsilon} (\rho^{\sigma})} + \| Y_{\varepsilon} \|_{C_T^{\beta / 2}
     L^{\infty, \varepsilon} (\rho^{\sigma})} \big) \| \llbracket
     X_{\varepsilon}^2 \rrbracket \|^2_{C \CC^{- 1 - \kappa, \varepsilon}
     (\rho^{\sigma})} \]
  \[ + (1+\lambda \| \llbracket X_{\varepsilon}^2 \rrbracket \|_{C_T \CC^{- 1 -
     \kappa, \varepsilon} (\rho^{\sigma})})^{6\kappa} \| \llbracket X_{\varepsilon}^2 \rrbracket \|^{2}_{C_T \CC^{- 1 -
     \kappa, \varepsilon} (\rho^{\sigma})} \| Y_{\varepsilon} \|_{C_T \CC^{1
     / 2 - \kappa, \varepsilon} (\rho^{\sigma})}\]
     \[ \lesssim (1+\lambda+\lambda|\log t| +\lambda^{2}) \|
\mathbb{X}_{\varepsilon}\|^{7+\vartheta} \]
  and the first claim follows since $\sigma > 0$ was chosen arbitrarily.
  
  Next, we recall \eqref{eq:YY} and the fact that
  $X_{\varepsilon}^{\!\resizebox{!}{.8em}{
\begin{tikzpicture}
\pgfpathmoveto{\pgfqpoint{0cm}{-0.035cm}}
\pgfpathlineto{\pgfqpoint{1.976cm}{-0.035cm}}
\pgfpathlineto{\pgfqpoint{1.976cm}{1.94cm}}
\pgfpathlineto{\pgfqpoint{0cm}{1.94cm}}
\pgfpathclose
\pgfusepath{clip}
\begin{pgfscope}
\begin{pgfscope}
\pgfpathmoveto{\pgfqpoint{0cm}{-0.035cm}}
\pgfpathlineto{\pgfqpoint{1.976cm}{-0.035cm}}
\pgfpathlineto{\pgfqpoint{1.976cm}{1.94cm}}
\pgfpathlineto{\pgfqpoint{0cm}{1.94cm}}
\pgfpathclose
\pgfusepath{clip}
\begin{pgfscope}
\begin{pgfscope}
\pgfsetdash{}{0cm}
\pgfsetlinewidth{0.818mm}
\pgfsetroundcap
\pgfsetroundjoin
\pgfsetmiterlimit{7.0}
\definecolor{eps2pgf_color}{gray}{0}\pgfsetstrokecolor{eps2pgf_color}\pgfsetfillcolor{eps2pgf_color}
\pgfpathmoveto{\pgfqpoint{0.117cm}{1.815cm}}
\pgfpathlineto{\pgfqpoint{0.682cm}{1.065cm}}
\pgfpathlineto{\pgfqpoint{1.246cm}{1.815cm}}
\pgfusepath{stroke}
\end{pgfscope}
\definecolor{eps2pgf_color}{gray}{0}\pgfsetstrokecolor{eps2pgf_color}\pgfsetfillcolor{eps2pgf_color}
\pgfpathmoveto{\pgfqpoint{0.273cm}{1.789cm}}
\pgfpathcurveto{\pgfqpoint{0.273cm}{1.825cm}}{\pgfqpoint{0.259cm}{1.86cm}}{\pgfqpoint{0.233cm}{1.886cm}}
\pgfpathcurveto{\pgfqpoint{0.207cm}{1.912cm}}{\pgfqpoint{0.173cm}{1.926cm}}{\pgfqpoint{0.137cm}{1.926cm}}
\pgfpathcurveto{\pgfqpoint{0.1cm}{1.926cm}}{\pgfqpoint{0.066cm}{1.912cm}}{\pgfqpoint{0.04cm}{1.886cm}}
\pgfpathcurveto{\pgfqpoint{0.014cm}{1.86cm}}{\pgfqpoint{0cm}{1.825cm}}{\pgfqpoint{0cm}{1.789cm}}
\pgfpathcurveto{\pgfqpoint{0cm}{1.753cm}}{\pgfqpoint{0.014cm}{1.718cm}}{\pgfqpoint{0.04cm}{1.692cm}}
\pgfpathcurveto{\pgfqpoint{0.066cm}{1.667cm}}{\pgfqpoint{0.1cm}{1.652cm}}{\pgfqpoint{0.137cm}{1.652cm}}
\pgfpathcurveto{\pgfqpoint{0.173cm}{1.652cm}}{\pgfqpoint{0.207cm}{1.667cm}}{\pgfqpoint{0.233cm}{1.692cm}}
\pgfpathcurveto{\pgfqpoint{0.259cm}{1.718cm}}{\pgfqpoint{0.273cm}{1.753cm}}{\pgfqpoint{0.273cm}{1.789cm}}
\pgfusepath{fill}
\begin{pgfscope}
\pgfsetdash{}{0cm}
\pgfsetlinewidth{0.818mm}
\pgfsetmiterlimit{7.0}
\pgfpathmoveto{\pgfqpoint{0.682cm}{1.065cm}}
\pgfpathlineto{\pgfqpoint{0.679cm}{1.812cm}}
\pgfusepath{stroke}
\end{pgfscope}
\pgfpathmoveto{\pgfqpoint{0.815cm}{1.793cm}}
\pgfpathcurveto{\pgfqpoint{0.815cm}{1.829cm}}{\pgfqpoint{0.801cm}{1.864cm}}{\pgfqpoint{0.775cm}{1.89cm}}
\pgfpathcurveto{\pgfqpoint{0.75cm}{1.915cm}}{\pgfqpoint{0.715cm}{1.93cm}}{\pgfqpoint{0.679cm}{1.93cm}}
\pgfpathcurveto{\pgfqpoint{0.643cm}{1.93cm}}{\pgfqpoint{0.608cm}{1.915cm}}{\pgfqpoint{0.582cm}{1.89cm}}
\pgfpathcurveto{\pgfqpoint{0.557cm}{1.864cm}}{\pgfqpoint{0.542cm}{1.829cm}}{\pgfqpoint{0.542cm}{1.793cm}}
\pgfpathcurveto{\pgfqpoint{0.542cm}{1.756cm}}{\pgfqpoint{0.557cm}{1.722cm}}{\pgfqpoint{0.582cm}{1.696cm}}
\pgfpathcurveto{\pgfqpoint{0.608cm}{1.67cm}}{\pgfqpoint{0.643cm}{1.656cm}}{\pgfqpoint{0.679cm}{1.656cm}}
\pgfpathcurveto{\pgfqpoint{0.715cm}{1.656cm}}{\pgfqpoint{0.75cm}{1.67cm}}{\pgfqpoint{0.775cm}{1.696cm}}
\pgfpathcurveto{\pgfqpoint{0.801cm}{1.722cm}}{\pgfqpoint{0.815cm}{1.756cm}}{\pgfqpoint{0.815cm}{1.793cm}}
\pgfusepath{fill}
\pgfpathmoveto{\pgfqpoint{1.345cm}{1.765cm}}
\pgfpathcurveto{\pgfqpoint{1.345cm}{1.801cm}}{\pgfqpoint{1.331cm}{1.836cm}}{\pgfqpoint{1.305cm}{1.862cm}}
\pgfpathcurveto{\pgfqpoint{1.28cm}{1.887cm}}{\pgfqpoint{1.245cm}{1.902cm}}{\pgfqpoint{1.209cm}{1.902cm}}
\pgfpathcurveto{\pgfqpoint{1.172cm}{1.902cm}}{\pgfqpoint{1.138cm}{1.887cm}}{\pgfqpoint{1.112cm}{1.862cm}}
\pgfpathcurveto{\pgfqpoint{1.087cm}{1.836cm}}{\pgfqpoint{1.072cm}{1.801cm}}{\pgfqpoint{1.072cm}{1.765cm}}
\pgfpathcurveto{\pgfqpoint{1.072cm}{1.728cm}}{\pgfqpoint{1.087cm}{1.694cm}}{\pgfqpoint{1.112cm}{1.668cm}}
\pgfpathcurveto{\pgfqpoint{1.138cm}{1.642cm}}{\pgfqpoint{1.172cm}{1.628cm}}{\pgfqpoint{1.209cm}{1.628cm}}
\pgfpathcurveto{\pgfqpoint{1.245cm}{1.628cm}}{\pgfqpoint{1.28cm}{1.642cm}}{\pgfqpoint{1.305cm}{1.668cm}}
\pgfpathcurveto{\pgfqpoint{1.331cm}{1.694cm}}{\pgfqpoint{1.345cm}{1.728cm}}{\pgfqpoint{1.345cm}{1.765cm}}
\pgfusepath{fill}
\begin{pgfscope}
\pgfsetdash{}{0cm}
\pgfsetlinewidth{0.818mm}
\pgfsetroundcap
\pgfsetroundjoin
\pgfsetmiterlimit{7.0}
\pgfpathmoveto{\pgfqpoint{0.682cm}{1.065cm}}
\pgfpathlineto{\pgfqpoint{1.246cm}{0.315cm}}
\pgfpathlineto{\pgfqpoint{1.811cm}{1.065cm}}
\pgfusepath{stroke}
\end{pgfscope}
\pgfpathmoveto{\pgfqpoint{1.948cm}{1.065cm}}
\pgfpathcurveto{\pgfqpoint{1.948cm}{1.101cm}}{\pgfqpoint{1.933cm}{1.136cm}}{\pgfqpoint{1.907cm}{1.162cm}}
\pgfpathcurveto{\pgfqpoint{1.882cm}{1.187cm}}{\pgfqpoint{1.847cm}{1.202cm}}{\pgfqpoint{1.811cm}{1.202cm}}
\pgfpathcurveto{\pgfqpoint{1.775cm}{1.202cm}}{\pgfqpoint{1.74cm}{1.187cm}}{\pgfqpoint{1.714cm}{1.162cm}}
\pgfpathcurveto{\pgfqpoint{1.689cm}{1.136cm}}{\pgfqpoint{1.674cm}{1.101cm}}{\pgfqpoint{1.674cm}{1.065cm}}
\pgfpathcurveto{\pgfqpoint{1.674cm}{1.029cm}}{\pgfqpoint{1.689cm}{0.994cm}}{\pgfqpoint{1.714cm}{0.968cm}}
\pgfpathcurveto{\pgfqpoint{1.74cm}{0.942cm}}{\pgfqpoint{1.775cm}{0.928cm}}{\pgfqpoint{1.811cm}{0.928cm}}
\pgfpathcurveto{\pgfqpoint{1.847cm}{0.928cm}}{\pgfqpoint{1.882cm}{0.942cm}}{\pgfqpoint{1.907cm}{0.968cm}}
\pgfpathcurveto{\pgfqpoint{1.933cm}{0.994cm}}{\pgfqpoint{1.948cm}{1.029cm}}{\pgfqpoint{1.948cm}{1.065cm}}
\pgfusepath{fill}
\begin{pgfscope}
\pgfsetdash{}{0cm}
\pgfsetlinewidth{0.818mm}
\pgfsetmiterlimit{4.0}
\pgfpathmoveto{\pgfqpoint{1.383cm}{0.178cm}}
\pgfpathcurveto{\pgfqpoint{1.383cm}{0.214cm}}{\pgfqpoint{1.369cm}{0.249cm}}{\pgfqpoint{1.343cm}{0.275cm}}
\pgfpathcurveto{\pgfqpoint{1.317cm}{0.3cm}}{\pgfqpoint{1.283cm}{0.315cm}}{\pgfqpoint{1.246cm}{0.315cm}}
\pgfpathcurveto{\pgfqpoint{1.21cm}{0.315cm}}{\pgfqpoint{1.175cm}{0.3cm}}{\pgfqpoint{1.15cm}{0.275cm}}
\pgfpathcurveto{\pgfqpoint{1.124cm}{0.249cm}}{\pgfqpoint{1.11cm}{0.214cm}}{\pgfqpoint{1.11cm}{0.178cm}}
\pgfpathcurveto{\pgfqpoint{1.11cm}{0.141cm}}{\pgfqpoint{1.124cm}{0.107cm}}{\pgfqpoint{1.15cm}{0.081cm}}
\pgfpathcurveto{\pgfqpoint{1.175cm}{0.055cm}}{\pgfqpoint{1.21cm}{0.041cm}}{\pgfqpoint{1.246cm}{0.041cm}}
\pgfpathcurveto{\pgfqpoint{1.283cm}{0.041cm}}{\pgfqpoint{1.317cm}{0.055cm}}{\pgfqpoint{1.343cm}{0.081cm}}
\pgfpathcurveto{\pgfqpoint{1.369cm}{0.107cm}}{\pgfqpoint{1.383cm}{0.141cm}}{\pgfqpoint{1.383cm}{0.178cm}}
\pgfusepath{stroke}
\end{pgfscope}
\end{pgfscope}
\end{pgfscope}
\end{pgfscope}
\end{tikzpicture}}} = X_{\varepsilon} \circ
  X_{\varepsilon}^{\!\resizebox{0.6em}{!}{
\begin{tikzpicture}
\pgfpathmoveto{\pgfqpoint{0cm}{-0.035cm}}
\pgfpathlineto{\pgfqpoint{1.376cm}{-0.035cm}}
\pgfpathlineto{\pgfqpoint{1.376cm}{1.552cm}}
\pgfpathlineto{\pgfqpoint{0cm}{1.552cm}}
\pgfpathclose
\pgfusepath{clip}
\begin{pgfscope}
\begin{pgfscope}
\pgfpathmoveto{\pgfqpoint{0cm}{-0.035cm}}
\pgfpathlineto{\pgfqpoint{1.376cm}{-0.035cm}}
\pgfpathlineto{\pgfqpoint{1.376cm}{1.552cm}}
\pgfpathlineto{\pgfqpoint{0cm}{1.552cm}}
\pgfpathclose
\pgfusepath{clip}
\begin{pgfscope}
\begin{pgfscope}
\pgfsetdash{}{0cm}
\pgfsetlinewidth{0.818mm}
\pgfsetroundcap
\pgfsetroundjoin
\pgfsetmiterlimit{7.0}
\definecolor{eps2pgf_color}{gray}{0}\pgfsetstrokecolor{eps2pgf_color}\pgfsetfillcolor{eps2pgf_color}
\pgfpathmoveto{\pgfqpoint{0.117cm}{1.421cm}}
\pgfpathlineto{\pgfqpoint{0.682cm}{0.671cm}}
\pgfpathlineto{\pgfqpoint{1.246cm}{1.421cm}}
\pgfusepath{stroke}
\end{pgfscope}
\definecolor{eps2pgf_color}{gray}{0}\pgfsetstrokecolor{eps2pgf_color}\pgfsetfillcolor{eps2pgf_color}
\pgfpathmoveto{\pgfqpoint{0.273cm}{1.395cm}}
\pgfpathcurveto{\pgfqpoint{0.273cm}{1.432cm}}{\pgfqpoint{0.259cm}{1.467cm}}{\pgfqpoint{0.233cm}{1.492cm}}
\pgfpathcurveto{\pgfqpoint{0.207cm}{1.518cm}}{\pgfqpoint{0.173cm}{1.532cm}}{\pgfqpoint{0.137cm}{1.532cm}}
\pgfpathcurveto{\pgfqpoint{0.1cm}{1.532cm}}{\pgfqpoint{0.066cm}{1.518cm}}{\pgfqpoint{0.04cm}{1.492cm}}
\pgfpathcurveto{\pgfqpoint{0.014cm}{1.467cm}}{\pgfqpoint{0cm}{1.432cm}}{\pgfqpoint{0cm}{1.395cm}}
\pgfpathcurveto{\pgfqpoint{0cm}{1.359cm}}{\pgfqpoint{0.014cm}{1.324cm}}{\pgfqpoint{0.04cm}{1.299cm}}
\pgfpathcurveto{\pgfqpoint{0.066cm}{1.273cm}}{\pgfqpoint{0.1cm}{1.258cm}}{\pgfqpoint{0.137cm}{1.258cm}}
\pgfpathcurveto{\pgfqpoint{0.173cm}{1.258cm}}{\pgfqpoint{0.207cm}{1.273cm}}{\pgfqpoint{0.233cm}{1.299cm}}
\pgfpathcurveto{\pgfqpoint{0.259cm}{1.324cm}}{\pgfqpoint{0.273cm}{1.359cm}}{\pgfqpoint{0.273cm}{1.395cm}}
\pgfusepath{fill}
\begin{pgfscope}
\pgfsetdash{}{0cm}
\pgfsetlinewidth{0.818mm}
\pgfsetmiterlimit{7.0}
\pgfpathmoveto{\pgfqpoint{0.682cm}{0.671cm}}
\pgfpathlineto{\pgfqpoint{0.679cm}{1.418cm}}
\pgfusepath{stroke}
\end{pgfscope}
\pgfpathmoveto{\pgfqpoint{0.815cm}{1.399cm}}
\pgfpathcurveto{\pgfqpoint{0.815cm}{1.435cm}}{\pgfqpoint{0.801cm}{1.47cm}}{\pgfqpoint{0.775cm}{1.496cm}}
\pgfpathcurveto{\pgfqpoint{0.75cm}{1.521cm}}{\pgfqpoint{0.715cm}{1.536cm}}{\pgfqpoint{0.679cm}{1.536cm}}
\pgfpathcurveto{\pgfqpoint{0.643cm}{1.536cm}}{\pgfqpoint{0.608cm}{1.521cm}}{\pgfqpoint{0.582cm}{1.496cm}}
\pgfpathcurveto{\pgfqpoint{0.557cm}{1.47cm}}{\pgfqpoint{0.542cm}{1.435cm}}{\pgfqpoint{0.542cm}{1.399cm}}
\pgfpathcurveto{\pgfqpoint{0.542cm}{1.363cm}}{\pgfqpoint{0.557cm}{1.328cm}}{\pgfqpoint{0.582cm}{1.302cm}}
\pgfpathcurveto{\pgfqpoint{0.608cm}{1.276cm}}{\pgfqpoint{0.643cm}{1.262cm}}{\pgfqpoint{0.679cm}{1.262cm}}
\pgfpathcurveto{\pgfqpoint{0.715cm}{1.262cm}}{\pgfqpoint{0.75cm}{1.276cm}}{\pgfqpoint{0.775cm}{1.302cm}}
\pgfpathcurveto{\pgfqpoint{0.801cm}{1.328cm}}{\pgfqpoint{0.815cm}{1.363cm}}{\pgfqpoint{0.815cm}{1.399cm}}
\pgfusepath{fill}
\pgfpathmoveto{\pgfqpoint{1.345cm}{1.371cm}}
\pgfpathcurveto{\pgfqpoint{1.345cm}{1.408cm}}{\pgfqpoint{1.331cm}{1.442cm}}{\pgfqpoint{1.305cm}{1.468cm}}
\pgfpathcurveto{\pgfqpoint{1.28cm}{1.494cm}}{\pgfqpoint{1.245cm}{1.508cm}}{\pgfqpoint{1.209cm}{1.508cm}}
\pgfpathcurveto{\pgfqpoint{1.172cm}{1.508cm}}{\pgfqpoint{1.138cm}{1.494cm}}{\pgfqpoint{1.112cm}{1.468cm}}
\pgfpathcurveto{\pgfqpoint{1.087cm}{1.442cm}}{\pgfqpoint{1.072cm}{1.408cm}}{\pgfqpoint{1.072cm}{1.371cm}}
\pgfpathcurveto{\pgfqpoint{1.072cm}{1.335cm}}{\pgfqpoint{1.087cm}{1.3cm}}{\pgfqpoint{1.112cm}{1.274cm}}
\pgfpathcurveto{\pgfqpoint{1.138cm}{1.249cm}}{\pgfqpoint{1.172cm}{1.234cm}}{\pgfqpoint{1.209cm}{1.234cm}}
\pgfpathcurveto{\pgfqpoint{1.245cm}{1.234cm}}{\pgfqpoint{1.28cm}{1.249cm}}{\pgfqpoint{1.305cm}{1.274cm}}
\pgfpathcurveto{\pgfqpoint{1.331cm}{1.3cm}}{\pgfqpoint{1.345cm}{1.335cm}}{\pgfqpoint{1.345cm}{1.371cm}}
\pgfusepath{fill}
\begin{pgfscope}
\pgfsetdash{}{0cm}
\pgfsetlinewidth{0.818mm}
\pgfsetroundcap
\pgfsetmiterlimit{4.0}
\pgfpathmoveto{\pgfqpoint{0.682cm}{0.671cm}}
\pgfpathlineto{\pgfqpoint{0.682cm}{0.042cm}}
\pgfusepath{stroke}
\end{pgfscope}
\end{pgfscope}
\end{pgfscope}
\end{pgfscope}
\end{tikzpicture}}}$ can be constructed without any renormalization
  in $C_T \CC^{- \kappa, \varepsilon} (\rho^{\sigma})$. As a consequence, the
  resonant term reads
  \begin{equation}\label{eq:XY-res3}
  X_{\varepsilon} \circ Y_{\varepsilon} = - \lambda X_{\varepsilon}^{\!\resizebox{!}{.8em}{
\begin{tikzpicture}
\pgfpathmoveto{\pgfqpoint{0cm}{-0.035cm}}
\pgfpathlineto{\pgfqpoint{1.976cm}{-0.035cm}}
\pgfpathlineto{\pgfqpoint{1.976cm}{1.94cm}}
\pgfpathlineto{\pgfqpoint{0cm}{1.94cm}}
\pgfpathclose
\pgfusepath{clip}
\begin{pgfscope}
\begin{pgfscope}
\pgfpathmoveto{\pgfqpoint{0cm}{-0.035cm}}
\pgfpathlineto{\pgfqpoint{1.976cm}{-0.035cm}}
\pgfpathlineto{\pgfqpoint{1.976cm}{1.94cm}}
\pgfpathlineto{\pgfqpoint{0cm}{1.94cm}}
\pgfpathclose
\pgfusepath{clip}
\begin{pgfscope}
\begin{pgfscope}
\pgfsetdash{}{0cm}
\pgfsetlinewidth{0.818mm}
\pgfsetroundcap
\pgfsetroundjoin
\pgfsetmiterlimit{7.0}
\definecolor{eps2pgf_color}{gray}{0}\pgfsetstrokecolor{eps2pgf_color}\pgfsetfillcolor{eps2pgf_color}
\pgfpathmoveto{\pgfqpoint{0.117cm}{1.815cm}}
\pgfpathlineto{\pgfqpoint{0.682cm}{1.065cm}}
\pgfpathlineto{\pgfqpoint{1.246cm}{1.815cm}}
\pgfusepath{stroke}
\end{pgfscope}
\definecolor{eps2pgf_color}{gray}{0}\pgfsetstrokecolor{eps2pgf_color}\pgfsetfillcolor{eps2pgf_color}
\pgfpathmoveto{\pgfqpoint{0.273cm}{1.789cm}}
\pgfpathcurveto{\pgfqpoint{0.273cm}{1.825cm}}{\pgfqpoint{0.259cm}{1.86cm}}{\pgfqpoint{0.233cm}{1.886cm}}
\pgfpathcurveto{\pgfqpoint{0.207cm}{1.912cm}}{\pgfqpoint{0.173cm}{1.926cm}}{\pgfqpoint{0.137cm}{1.926cm}}
\pgfpathcurveto{\pgfqpoint{0.1cm}{1.926cm}}{\pgfqpoint{0.066cm}{1.912cm}}{\pgfqpoint{0.04cm}{1.886cm}}
\pgfpathcurveto{\pgfqpoint{0.014cm}{1.86cm}}{\pgfqpoint{0cm}{1.825cm}}{\pgfqpoint{0cm}{1.789cm}}
\pgfpathcurveto{\pgfqpoint{0cm}{1.753cm}}{\pgfqpoint{0.014cm}{1.718cm}}{\pgfqpoint{0.04cm}{1.692cm}}
\pgfpathcurveto{\pgfqpoint{0.066cm}{1.667cm}}{\pgfqpoint{0.1cm}{1.652cm}}{\pgfqpoint{0.137cm}{1.652cm}}
\pgfpathcurveto{\pgfqpoint{0.173cm}{1.652cm}}{\pgfqpoint{0.207cm}{1.667cm}}{\pgfqpoint{0.233cm}{1.692cm}}
\pgfpathcurveto{\pgfqpoint{0.259cm}{1.718cm}}{\pgfqpoint{0.273cm}{1.753cm}}{\pgfqpoint{0.273cm}{1.789cm}}
\pgfusepath{fill}
\begin{pgfscope}
\pgfsetdash{}{0cm}
\pgfsetlinewidth{0.818mm}
\pgfsetmiterlimit{7.0}
\pgfpathmoveto{\pgfqpoint{0.682cm}{1.065cm}}
\pgfpathlineto{\pgfqpoint{0.679cm}{1.812cm}}
\pgfusepath{stroke}
\end{pgfscope}
\pgfpathmoveto{\pgfqpoint{0.815cm}{1.793cm}}
\pgfpathcurveto{\pgfqpoint{0.815cm}{1.829cm}}{\pgfqpoint{0.801cm}{1.864cm}}{\pgfqpoint{0.775cm}{1.89cm}}
\pgfpathcurveto{\pgfqpoint{0.75cm}{1.915cm}}{\pgfqpoint{0.715cm}{1.93cm}}{\pgfqpoint{0.679cm}{1.93cm}}
\pgfpathcurveto{\pgfqpoint{0.643cm}{1.93cm}}{\pgfqpoint{0.608cm}{1.915cm}}{\pgfqpoint{0.582cm}{1.89cm}}
\pgfpathcurveto{\pgfqpoint{0.557cm}{1.864cm}}{\pgfqpoint{0.542cm}{1.829cm}}{\pgfqpoint{0.542cm}{1.793cm}}
\pgfpathcurveto{\pgfqpoint{0.542cm}{1.756cm}}{\pgfqpoint{0.557cm}{1.722cm}}{\pgfqpoint{0.582cm}{1.696cm}}
\pgfpathcurveto{\pgfqpoint{0.608cm}{1.67cm}}{\pgfqpoint{0.643cm}{1.656cm}}{\pgfqpoint{0.679cm}{1.656cm}}
\pgfpathcurveto{\pgfqpoint{0.715cm}{1.656cm}}{\pgfqpoint{0.75cm}{1.67cm}}{\pgfqpoint{0.775cm}{1.696cm}}
\pgfpathcurveto{\pgfqpoint{0.801cm}{1.722cm}}{\pgfqpoint{0.815cm}{1.756cm}}{\pgfqpoint{0.815cm}{1.793cm}}
\pgfusepath{fill}
\pgfpathmoveto{\pgfqpoint{1.345cm}{1.765cm}}
\pgfpathcurveto{\pgfqpoint{1.345cm}{1.801cm}}{\pgfqpoint{1.331cm}{1.836cm}}{\pgfqpoint{1.305cm}{1.862cm}}
\pgfpathcurveto{\pgfqpoint{1.28cm}{1.887cm}}{\pgfqpoint{1.245cm}{1.902cm}}{\pgfqpoint{1.209cm}{1.902cm}}
\pgfpathcurveto{\pgfqpoint{1.172cm}{1.902cm}}{\pgfqpoint{1.138cm}{1.887cm}}{\pgfqpoint{1.112cm}{1.862cm}}
\pgfpathcurveto{\pgfqpoint{1.087cm}{1.836cm}}{\pgfqpoint{1.072cm}{1.801cm}}{\pgfqpoint{1.072cm}{1.765cm}}
\pgfpathcurveto{\pgfqpoint{1.072cm}{1.728cm}}{\pgfqpoint{1.087cm}{1.694cm}}{\pgfqpoint{1.112cm}{1.668cm}}
\pgfpathcurveto{\pgfqpoint{1.138cm}{1.642cm}}{\pgfqpoint{1.172cm}{1.628cm}}{\pgfqpoint{1.209cm}{1.628cm}}
\pgfpathcurveto{\pgfqpoint{1.245cm}{1.628cm}}{\pgfqpoint{1.28cm}{1.642cm}}{\pgfqpoint{1.305cm}{1.668cm}}
\pgfpathcurveto{\pgfqpoint{1.331cm}{1.694cm}}{\pgfqpoint{1.345cm}{1.728cm}}{\pgfqpoint{1.345cm}{1.765cm}}
\pgfusepath{fill}
\begin{pgfscope}
\pgfsetdash{}{0cm}
\pgfsetlinewidth{0.818mm}
\pgfsetroundcap
\pgfsetroundjoin
\pgfsetmiterlimit{7.0}
\pgfpathmoveto{\pgfqpoint{0.682cm}{1.065cm}}
\pgfpathlineto{\pgfqpoint{1.246cm}{0.315cm}}
\pgfpathlineto{\pgfqpoint{1.811cm}{1.065cm}}
\pgfusepath{stroke}
\end{pgfscope}
\pgfpathmoveto{\pgfqpoint{1.948cm}{1.065cm}}
\pgfpathcurveto{\pgfqpoint{1.948cm}{1.101cm}}{\pgfqpoint{1.933cm}{1.136cm}}{\pgfqpoint{1.907cm}{1.162cm}}
\pgfpathcurveto{\pgfqpoint{1.882cm}{1.187cm}}{\pgfqpoint{1.847cm}{1.202cm}}{\pgfqpoint{1.811cm}{1.202cm}}
\pgfpathcurveto{\pgfqpoint{1.775cm}{1.202cm}}{\pgfqpoint{1.74cm}{1.187cm}}{\pgfqpoint{1.714cm}{1.162cm}}
\pgfpathcurveto{\pgfqpoint{1.689cm}{1.136cm}}{\pgfqpoint{1.674cm}{1.101cm}}{\pgfqpoint{1.674cm}{1.065cm}}
\pgfpathcurveto{\pgfqpoint{1.674cm}{1.029cm}}{\pgfqpoint{1.689cm}{0.994cm}}{\pgfqpoint{1.714cm}{0.968cm}}
\pgfpathcurveto{\pgfqpoint{1.74cm}{0.942cm}}{\pgfqpoint{1.775cm}{0.928cm}}{\pgfqpoint{1.811cm}{0.928cm}}
\pgfpathcurveto{\pgfqpoint{1.847cm}{0.928cm}}{\pgfqpoint{1.882cm}{0.942cm}}{\pgfqpoint{1.907cm}{0.968cm}}
\pgfpathcurveto{\pgfqpoint{1.933cm}{0.994cm}}{\pgfqpoint{1.948cm}{1.029cm}}{\pgfqpoint{1.948cm}{1.065cm}}
\pgfusepath{fill}
\begin{pgfscope}
\pgfsetdash{}{0cm}
\pgfsetlinewidth{0.818mm}
\pgfsetmiterlimit{4.0}
\pgfpathmoveto{\pgfqpoint{1.383cm}{0.178cm}}
\pgfpathcurveto{\pgfqpoint{1.383cm}{0.214cm}}{\pgfqpoint{1.369cm}{0.249cm}}{\pgfqpoint{1.343cm}{0.275cm}}
\pgfpathcurveto{\pgfqpoint{1.317cm}{0.3cm}}{\pgfqpoint{1.283cm}{0.315cm}}{\pgfqpoint{1.246cm}{0.315cm}}
\pgfpathcurveto{\pgfqpoint{1.21cm}{0.315cm}}{\pgfqpoint{1.175cm}{0.3cm}}{\pgfqpoint{1.15cm}{0.275cm}}
\pgfpathcurveto{\pgfqpoint{1.124cm}{0.249cm}}{\pgfqpoint{1.11cm}{0.214cm}}{\pgfqpoint{1.11cm}{0.178cm}}
\pgfpathcurveto{\pgfqpoint{1.11cm}{0.141cm}}{\pgfqpoint{1.124cm}{0.107cm}}{\pgfqpoint{1.15cm}{0.081cm}}
\pgfpathcurveto{\pgfqpoint{1.175cm}{0.055cm}}{\pgfqpoint{1.21cm}{0.041cm}}{\pgfqpoint{1.246cm}{0.041cm}}
\pgfpathcurveto{\pgfqpoint{1.283cm}{0.041cm}}{\pgfqpoint{1.317cm}{0.055cm}}{\pgfqpoint{1.343cm}{0.081cm}}
\pgfpathcurveto{\pgfqpoint{1.369cm}{0.107cm}}{\pgfqpoint{1.383cm}{0.141cm}}{\pgfqpoint{1.383cm}{0.178cm}}
\pgfusepath{stroke}
\end{pgfscope}
\end{pgfscope}
\end{pgfscope}
\end{pgfscope}
\end{tikzpicture}}} -
     X_{\varepsilon} \circ \LL_{\varepsilon}^{- 1} \left[ 3\lambda\left(
     \UU^{\varepsilon}_{>} \llbracket X_{\varepsilon}^2 \rrbracket \right)
     \succ Y_{\varepsilon} \right],
  \end{equation}
  where the for the second term we have (since $\UU^{\varepsilon}_{>}$ is a contraction) that
  \[ \lambda \left\| X_{\varepsilon} \circ \LL_{\varepsilon}^{- 1} \left[ 3 \left(
     \UU^{\varepsilon}_{>} \llbracket X_{\varepsilon}^2 \rrbracket \right)
     \succ Y_{\varepsilon} \right] \right\|_{C_T \CC^{1 / 2 - 2 \kappa,
     \varepsilon} (\rho^{3 \sigma})} \]
  \[ \lesssim \lambda \| X_{\varepsilon} \|_{C_T \CC^{- 1 / 2 - \kappa, \varepsilon}
     (\rho^{\sigma})} \left\| \left( \UU^{\varepsilon}_{>} \llbracket
     X_{\varepsilon}^2 \rrbracket \right) \succ Y_{\varepsilon} \right\|_{C_T
     \CC^{- 1 - \kappa, \varepsilon} (\rho^{2 \sigma})} \]
  \begin{equation}\label{eq:XY-res}
  \lesssim \lambda \| X_{\varepsilon} \|_{C_T \CC^{- 1 / 2 - \kappa, \varepsilon}
     (\rho^{\sigma})} \| \llbracket X_{\varepsilon}^2 \rrbracket \|_{C_T
     \CC^{- 1 - \kappa, \varepsilon} (\rho^{\sigma})} \| Y_{\varepsilon}
     \|_{C_T L^{\infty, \varepsilon} (\rho^{\sigma})} \lesssim \lambda^2 \|\mathbb{X}_{\varepsilon}\|^{6}.
     \end{equation}
  For the two paraproducts we obtain directly
  \begin{equation}\label{eq:XY-par1}
  \| X_{\varepsilon} \prec Y_{\varepsilon} \|_{C_T \CC^{- 2 \kappa,
     \varepsilon} (\rho^{3 \sigma})} \lesssim \| X_{\varepsilon} \|_{C_T
     \CC^{- 1 / 2 - \kappa, \varepsilon} (\rho^{\sigma})} \| Y_{\varepsilon}
     \|_{C_T \CC^{1 / 2 - \kappa, \varepsilon} (\rho^{\sigma})} \lesssim \lambda
\|\mathbb{X}_{\varepsilon}\|^{4},
\end{equation}
    \begin{equation}\label{eq:XY-par2}
     \| X_{\varepsilon} \succ Y_{\varepsilon} \|_{C_T \CC^{- 1 / 2 - \kappa,
     \varepsilon} (\rho^{3 \sigma})} \lesssim \| X_{\varepsilon} \|_{C_T
     \CC^{- 1 / 2 - \kappa, \varepsilon} (\rho^{\sigma})} \| Y_{\varepsilon}
     \|_{C_T L^{\infty, \varepsilon} (\rho^{\sigma})} \lesssim \lambda \|
     \mathbb{X}_{\varepsilon}\|^{4} .
     \end{equation}
  We proceed similarly for the remaining term, which is quadratic in
  $Y_{\varepsilon}$. We have
  \[ X_{\varepsilon} \circ Y_{\varepsilon}^2 = X_{\varepsilon} \circ (2
     Y_{\varepsilon} \prec Y_{\varepsilon}) + X_{\varepsilon} \circ
     (Y_{\varepsilon} \circ Y_{\varepsilon}) \]
  \[ = - X_{\varepsilon} \circ ( 2 Y_{\varepsilon} \prec \lambda X^{\!\resizebox{0.6em}{!}{
\begin{tikzpicture}
\pgfpathmoveto{\pgfqpoint{0cm}{-0.035cm}}
\pgfpathlineto{\pgfqpoint{1.376cm}{-0.035cm}}
\pgfpathlineto{\pgfqpoint{1.376cm}{1.552cm}}
\pgfpathlineto{\pgfqpoint{0cm}{1.552cm}}
\pgfpathclose
\pgfusepath{clip}
\begin{pgfscope}
\begin{pgfscope}
\pgfpathmoveto{\pgfqpoint{0cm}{-0.035cm}}
\pgfpathlineto{\pgfqpoint{1.376cm}{-0.035cm}}
\pgfpathlineto{\pgfqpoint{1.376cm}{1.552cm}}
\pgfpathlineto{\pgfqpoint{0cm}{1.552cm}}
\pgfpathclose
\pgfusepath{clip}
\begin{pgfscope}
\begin{pgfscope}
\pgfsetdash{}{0cm}
\pgfsetlinewidth{0.818mm}
\pgfsetroundcap
\pgfsetroundjoin
\pgfsetmiterlimit{7.0}
\definecolor{eps2pgf_color}{gray}{0}\pgfsetstrokecolor{eps2pgf_color}\pgfsetfillcolor{eps2pgf_color}
\pgfpathmoveto{\pgfqpoint{0.117cm}{1.421cm}}
\pgfpathlineto{\pgfqpoint{0.682cm}{0.671cm}}
\pgfpathlineto{\pgfqpoint{1.246cm}{1.421cm}}
\pgfusepath{stroke}
\end{pgfscope}
\definecolor{eps2pgf_color}{gray}{0}\pgfsetstrokecolor{eps2pgf_color}\pgfsetfillcolor{eps2pgf_color}
\pgfpathmoveto{\pgfqpoint{0.273cm}{1.395cm}}
\pgfpathcurveto{\pgfqpoint{0.273cm}{1.432cm}}{\pgfqpoint{0.259cm}{1.467cm}}{\pgfqpoint{0.233cm}{1.492cm}}
\pgfpathcurveto{\pgfqpoint{0.207cm}{1.518cm}}{\pgfqpoint{0.173cm}{1.532cm}}{\pgfqpoint{0.137cm}{1.532cm}}
\pgfpathcurveto{\pgfqpoint{0.1cm}{1.532cm}}{\pgfqpoint{0.066cm}{1.518cm}}{\pgfqpoint{0.04cm}{1.492cm}}
\pgfpathcurveto{\pgfqpoint{0.014cm}{1.467cm}}{\pgfqpoint{0cm}{1.432cm}}{\pgfqpoint{0cm}{1.395cm}}
\pgfpathcurveto{\pgfqpoint{0cm}{1.359cm}}{\pgfqpoint{0.014cm}{1.324cm}}{\pgfqpoint{0.04cm}{1.299cm}}
\pgfpathcurveto{\pgfqpoint{0.066cm}{1.273cm}}{\pgfqpoint{0.1cm}{1.258cm}}{\pgfqpoint{0.137cm}{1.258cm}}
\pgfpathcurveto{\pgfqpoint{0.173cm}{1.258cm}}{\pgfqpoint{0.207cm}{1.273cm}}{\pgfqpoint{0.233cm}{1.299cm}}
\pgfpathcurveto{\pgfqpoint{0.259cm}{1.324cm}}{\pgfqpoint{0.273cm}{1.359cm}}{\pgfqpoint{0.273cm}{1.395cm}}
\pgfusepath{fill}
\begin{pgfscope}
\pgfsetdash{}{0cm}
\pgfsetlinewidth{0.818mm}
\pgfsetmiterlimit{7.0}
\pgfpathmoveto{\pgfqpoint{0.682cm}{0.671cm}}
\pgfpathlineto{\pgfqpoint{0.679cm}{1.418cm}}
\pgfusepath{stroke}
\end{pgfscope}
\pgfpathmoveto{\pgfqpoint{0.815cm}{1.399cm}}
\pgfpathcurveto{\pgfqpoint{0.815cm}{1.435cm}}{\pgfqpoint{0.801cm}{1.47cm}}{\pgfqpoint{0.775cm}{1.496cm}}
\pgfpathcurveto{\pgfqpoint{0.75cm}{1.521cm}}{\pgfqpoint{0.715cm}{1.536cm}}{\pgfqpoint{0.679cm}{1.536cm}}
\pgfpathcurveto{\pgfqpoint{0.643cm}{1.536cm}}{\pgfqpoint{0.608cm}{1.521cm}}{\pgfqpoint{0.582cm}{1.496cm}}
\pgfpathcurveto{\pgfqpoint{0.557cm}{1.47cm}}{\pgfqpoint{0.542cm}{1.435cm}}{\pgfqpoint{0.542cm}{1.399cm}}
\pgfpathcurveto{\pgfqpoint{0.542cm}{1.363cm}}{\pgfqpoint{0.557cm}{1.328cm}}{\pgfqpoint{0.582cm}{1.302cm}}
\pgfpathcurveto{\pgfqpoint{0.608cm}{1.276cm}}{\pgfqpoint{0.643cm}{1.262cm}}{\pgfqpoint{0.679cm}{1.262cm}}
\pgfpathcurveto{\pgfqpoint{0.715cm}{1.262cm}}{\pgfqpoint{0.75cm}{1.276cm}}{\pgfqpoint{0.775cm}{1.302cm}}
\pgfpathcurveto{\pgfqpoint{0.801cm}{1.328cm}}{\pgfqpoint{0.815cm}{1.363cm}}{\pgfqpoint{0.815cm}{1.399cm}}
\pgfusepath{fill}
\pgfpathmoveto{\pgfqpoint{1.345cm}{1.371cm}}
\pgfpathcurveto{\pgfqpoint{1.345cm}{1.408cm}}{\pgfqpoint{1.331cm}{1.442cm}}{\pgfqpoint{1.305cm}{1.468cm}}
\pgfpathcurveto{\pgfqpoint{1.28cm}{1.494cm}}{\pgfqpoint{1.245cm}{1.508cm}}{\pgfqpoint{1.209cm}{1.508cm}}
\pgfpathcurveto{\pgfqpoint{1.172cm}{1.508cm}}{\pgfqpoint{1.138cm}{1.494cm}}{\pgfqpoint{1.112cm}{1.468cm}}
\pgfpathcurveto{\pgfqpoint{1.087cm}{1.442cm}}{\pgfqpoint{1.072cm}{1.408cm}}{\pgfqpoint{1.072cm}{1.371cm}}
\pgfpathcurveto{\pgfqpoint{1.072cm}{1.335cm}}{\pgfqpoint{1.087cm}{1.3cm}}{\pgfqpoint{1.112cm}{1.274cm}}
\pgfpathcurveto{\pgfqpoint{1.138cm}{1.249cm}}{\pgfqpoint{1.172cm}{1.234cm}}{\pgfqpoint{1.209cm}{1.234cm}}
\pgfpathcurveto{\pgfqpoint{1.245cm}{1.234cm}}{\pgfqpoint{1.28cm}{1.249cm}}{\pgfqpoint{1.305cm}{1.274cm}}
\pgfpathcurveto{\pgfqpoint{1.331cm}{1.3cm}}{\pgfqpoint{1.345cm}{1.335cm}}{\pgfqpoint{1.345cm}{1.371cm}}
\pgfusepath{fill}
\begin{pgfscope}
\pgfsetdash{}{0cm}
\pgfsetlinewidth{0.818mm}
\pgfsetroundcap
\pgfsetmiterlimit{4.0}
\pgfpathmoveto{\pgfqpoint{0.682cm}{0.671cm}}
\pgfpathlineto{\pgfqpoint{0.682cm}{0.042cm}}
\pgfusepath{stroke}
\end{pgfscope}
\end{pgfscope}
\end{pgfscope}
\end{pgfscope}
\end{tikzpicture}}}
     ) - X_{\varepsilon} \circ \left( 2 Y_{\varepsilon} \prec
     \LL_{\varepsilon}^{- 1} \left[ 3 \lambda \left( \UU_{>} \llbracket
     X_{\varepsilon}^2 \rrbracket \right) \succ Y_{\varepsilon} \right]
     \right) + X_{\varepsilon} \circ (Y_{\varepsilon} \circ Y_{\varepsilon})
  \]
  \[ = - 2 \lambda X_{\varepsilon}^{\!\resizebox{!}{.8em}{
\begin{tikzpicture}
\pgfpathmoveto{\pgfqpoint{0cm}{-0.035cm}}
\pgfpathlineto{\pgfqpoint{1.976cm}{-0.035cm}}
\pgfpathlineto{\pgfqpoint{1.976cm}{1.94cm}}
\pgfpathlineto{\pgfqpoint{0cm}{1.94cm}}
\pgfpathclose
\pgfusepath{clip}
\begin{pgfscope}
\begin{pgfscope}
\pgfpathmoveto{\pgfqpoint{0cm}{-0.035cm}}
\pgfpathlineto{\pgfqpoint{1.976cm}{-0.035cm}}
\pgfpathlineto{\pgfqpoint{1.976cm}{1.94cm}}
\pgfpathlineto{\pgfqpoint{0cm}{1.94cm}}
\pgfpathclose
\pgfusepath{clip}
\begin{pgfscope}
\begin{pgfscope}
\pgfsetdash{}{0cm}
\pgfsetlinewidth{0.818mm}
\pgfsetroundcap
\pgfsetroundjoin
\pgfsetmiterlimit{7.0}
\definecolor{eps2pgf_color}{gray}{0}\pgfsetstrokecolor{eps2pgf_color}\pgfsetfillcolor{eps2pgf_color}
\pgfpathmoveto{\pgfqpoint{0.117cm}{1.815cm}}
\pgfpathlineto{\pgfqpoint{0.682cm}{1.065cm}}
\pgfpathlineto{\pgfqpoint{1.246cm}{1.815cm}}
\pgfusepath{stroke}
\end{pgfscope}
\definecolor{eps2pgf_color}{gray}{0}\pgfsetstrokecolor{eps2pgf_color}\pgfsetfillcolor{eps2pgf_color}
\pgfpathmoveto{\pgfqpoint{0.273cm}{1.789cm}}
\pgfpathcurveto{\pgfqpoint{0.273cm}{1.825cm}}{\pgfqpoint{0.259cm}{1.86cm}}{\pgfqpoint{0.233cm}{1.886cm}}
\pgfpathcurveto{\pgfqpoint{0.207cm}{1.912cm}}{\pgfqpoint{0.173cm}{1.926cm}}{\pgfqpoint{0.137cm}{1.926cm}}
\pgfpathcurveto{\pgfqpoint{0.1cm}{1.926cm}}{\pgfqpoint{0.066cm}{1.912cm}}{\pgfqpoint{0.04cm}{1.886cm}}
\pgfpathcurveto{\pgfqpoint{0.014cm}{1.86cm}}{\pgfqpoint{0cm}{1.825cm}}{\pgfqpoint{0cm}{1.789cm}}
\pgfpathcurveto{\pgfqpoint{0cm}{1.753cm}}{\pgfqpoint{0.014cm}{1.718cm}}{\pgfqpoint{0.04cm}{1.692cm}}
\pgfpathcurveto{\pgfqpoint{0.066cm}{1.667cm}}{\pgfqpoint{0.1cm}{1.652cm}}{\pgfqpoint{0.137cm}{1.652cm}}
\pgfpathcurveto{\pgfqpoint{0.173cm}{1.652cm}}{\pgfqpoint{0.207cm}{1.667cm}}{\pgfqpoint{0.233cm}{1.692cm}}
\pgfpathcurveto{\pgfqpoint{0.259cm}{1.718cm}}{\pgfqpoint{0.273cm}{1.753cm}}{\pgfqpoint{0.273cm}{1.789cm}}
\pgfusepath{fill}
\begin{pgfscope}
\pgfsetdash{}{0cm}
\pgfsetlinewidth{0.818mm}
\pgfsetmiterlimit{7.0}
\pgfpathmoveto{\pgfqpoint{0.682cm}{1.065cm}}
\pgfpathlineto{\pgfqpoint{0.679cm}{1.812cm}}
\pgfusepath{stroke}
\end{pgfscope}
\pgfpathmoveto{\pgfqpoint{0.815cm}{1.793cm}}
\pgfpathcurveto{\pgfqpoint{0.815cm}{1.829cm}}{\pgfqpoint{0.801cm}{1.864cm}}{\pgfqpoint{0.775cm}{1.89cm}}
\pgfpathcurveto{\pgfqpoint{0.75cm}{1.915cm}}{\pgfqpoint{0.715cm}{1.93cm}}{\pgfqpoint{0.679cm}{1.93cm}}
\pgfpathcurveto{\pgfqpoint{0.643cm}{1.93cm}}{\pgfqpoint{0.608cm}{1.915cm}}{\pgfqpoint{0.582cm}{1.89cm}}
\pgfpathcurveto{\pgfqpoint{0.557cm}{1.864cm}}{\pgfqpoint{0.542cm}{1.829cm}}{\pgfqpoint{0.542cm}{1.793cm}}
\pgfpathcurveto{\pgfqpoint{0.542cm}{1.756cm}}{\pgfqpoint{0.557cm}{1.722cm}}{\pgfqpoint{0.582cm}{1.696cm}}
\pgfpathcurveto{\pgfqpoint{0.608cm}{1.67cm}}{\pgfqpoint{0.643cm}{1.656cm}}{\pgfqpoint{0.679cm}{1.656cm}}
\pgfpathcurveto{\pgfqpoint{0.715cm}{1.656cm}}{\pgfqpoint{0.75cm}{1.67cm}}{\pgfqpoint{0.775cm}{1.696cm}}
\pgfpathcurveto{\pgfqpoint{0.801cm}{1.722cm}}{\pgfqpoint{0.815cm}{1.756cm}}{\pgfqpoint{0.815cm}{1.793cm}}
\pgfusepath{fill}
\pgfpathmoveto{\pgfqpoint{1.345cm}{1.765cm}}
\pgfpathcurveto{\pgfqpoint{1.345cm}{1.801cm}}{\pgfqpoint{1.331cm}{1.836cm}}{\pgfqpoint{1.305cm}{1.862cm}}
\pgfpathcurveto{\pgfqpoint{1.28cm}{1.887cm}}{\pgfqpoint{1.245cm}{1.902cm}}{\pgfqpoint{1.209cm}{1.902cm}}
\pgfpathcurveto{\pgfqpoint{1.172cm}{1.902cm}}{\pgfqpoint{1.138cm}{1.887cm}}{\pgfqpoint{1.112cm}{1.862cm}}
\pgfpathcurveto{\pgfqpoint{1.087cm}{1.836cm}}{\pgfqpoint{1.072cm}{1.801cm}}{\pgfqpoint{1.072cm}{1.765cm}}
\pgfpathcurveto{\pgfqpoint{1.072cm}{1.728cm}}{\pgfqpoint{1.087cm}{1.694cm}}{\pgfqpoint{1.112cm}{1.668cm}}
\pgfpathcurveto{\pgfqpoint{1.138cm}{1.642cm}}{\pgfqpoint{1.172cm}{1.628cm}}{\pgfqpoint{1.209cm}{1.628cm}}
\pgfpathcurveto{\pgfqpoint{1.245cm}{1.628cm}}{\pgfqpoint{1.28cm}{1.642cm}}{\pgfqpoint{1.305cm}{1.668cm}}
\pgfpathcurveto{\pgfqpoint{1.331cm}{1.694cm}}{\pgfqpoint{1.345cm}{1.728cm}}{\pgfqpoint{1.345cm}{1.765cm}}
\pgfusepath{fill}
\begin{pgfscope}
\pgfsetdash{}{0cm}
\pgfsetlinewidth{0.818mm}
\pgfsetroundcap
\pgfsetroundjoin
\pgfsetmiterlimit{7.0}
\pgfpathmoveto{\pgfqpoint{0.682cm}{1.065cm}}
\pgfpathlineto{\pgfqpoint{1.246cm}{0.315cm}}
\pgfpathlineto{\pgfqpoint{1.811cm}{1.065cm}}
\pgfusepath{stroke}
\end{pgfscope}
\pgfpathmoveto{\pgfqpoint{1.948cm}{1.065cm}}
\pgfpathcurveto{\pgfqpoint{1.948cm}{1.101cm}}{\pgfqpoint{1.933cm}{1.136cm}}{\pgfqpoint{1.907cm}{1.162cm}}
\pgfpathcurveto{\pgfqpoint{1.882cm}{1.187cm}}{\pgfqpoint{1.847cm}{1.202cm}}{\pgfqpoint{1.811cm}{1.202cm}}
\pgfpathcurveto{\pgfqpoint{1.775cm}{1.202cm}}{\pgfqpoint{1.74cm}{1.187cm}}{\pgfqpoint{1.714cm}{1.162cm}}
\pgfpathcurveto{\pgfqpoint{1.689cm}{1.136cm}}{\pgfqpoint{1.674cm}{1.101cm}}{\pgfqpoint{1.674cm}{1.065cm}}
\pgfpathcurveto{\pgfqpoint{1.674cm}{1.029cm}}{\pgfqpoint{1.689cm}{0.994cm}}{\pgfqpoint{1.714cm}{0.968cm}}
\pgfpathcurveto{\pgfqpoint{1.74cm}{0.942cm}}{\pgfqpoint{1.775cm}{0.928cm}}{\pgfqpoint{1.811cm}{0.928cm}}
\pgfpathcurveto{\pgfqpoint{1.847cm}{0.928cm}}{\pgfqpoint{1.882cm}{0.942cm}}{\pgfqpoint{1.907cm}{0.968cm}}
\pgfpathcurveto{\pgfqpoint{1.933cm}{0.994cm}}{\pgfqpoint{1.948cm}{1.029cm}}{\pgfqpoint{1.948cm}{1.065cm}}
\pgfusepath{fill}
\begin{pgfscope}
\pgfsetdash{}{0cm}
\pgfsetlinewidth{0.818mm}
\pgfsetmiterlimit{4.0}
\pgfpathmoveto{\pgfqpoint{1.383cm}{0.178cm}}
\pgfpathcurveto{\pgfqpoint{1.383cm}{0.214cm}}{\pgfqpoint{1.369cm}{0.249cm}}{\pgfqpoint{1.343cm}{0.275cm}}
\pgfpathcurveto{\pgfqpoint{1.317cm}{0.3cm}}{\pgfqpoint{1.283cm}{0.315cm}}{\pgfqpoint{1.246cm}{0.315cm}}
\pgfpathcurveto{\pgfqpoint{1.21cm}{0.315cm}}{\pgfqpoint{1.175cm}{0.3cm}}{\pgfqpoint{1.15cm}{0.275cm}}
\pgfpathcurveto{\pgfqpoint{1.124cm}{0.249cm}}{\pgfqpoint{1.11cm}{0.214cm}}{\pgfqpoint{1.11cm}{0.178cm}}
\pgfpathcurveto{\pgfqpoint{1.11cm}{0.141cm}}{\pgfqpoint{1.124cm}{0.107cm}}{\pgfqpoint{1.15cm}{0.081cm}}
\pgfpathcurveto{\pgfqpoint{1.175cm}{0.055cm}}{\pgfqpoint{1.21cm}{0.041cm}}{\pgfqpoint{1.246cm}{0.041cm}}
\pgfpathcurveto{\pgfqpoint{1.283cm}{0.041cm}}{\pgfqpoint{1.317cm}{0.055cm}}{\pgfqpoint{1.343cm}{0.081cm}}
\pgfpathcurveto{\pgfqpoint{1.369cm}{0.107cm}}{\pgfqpoint{1.383cm}{0.141cm}}{\pgfqpoint{1.383cm}{0.178cm}}
\pgfusepath{stroke}
\end{pgfscope}
\end{pgfscope}
\end{pgfscope}
\end{pgfscope}
\end{tikzpicture}}} Y_{\varepsilon} - \lambda C_{\varepsilon}
     ( Y_{\varepsilon}, 2 X_{\varepsilon}^{\!\resizebox{0.6em}{!}{
\begin{tikzpicture}
\pgfpathmoveto{\pgfqpoint{0cm}{-0.035cm}}
\pgfpathlineto{\pgfqpoint{1.376cm}{-0.035cm}}
\pgfpathlineto{\pgfqpoint{1.376cm}{1.552cm}}
\pgfpathlineto{\pgfqpoint{0cm}{1.552cm}}
\pgfpathclose
\pgfusepath{clip}
\begin{pgfscope}
\begin{pgfscope}
\pgfpathmoveto{\pgfqpoint{0cm}{-0.035cm}}
\pgfpathlineto{\pgfqpoint{1.376cm}{-0.035cm}}
\pgfpathlineto{\pgfqpoint{1.376cm}{1.552cm}}
\pgfpathlineto{\pgfqpoint{0cm}{1.552cm}}
\pgfpathclose
\pgfusepath{clip}
\begin{pgfscope}
\begin{pgfscope}
\pgfsetdash{}{0cm}
\pgfsetlinewidth{0.818mm}
\pgfsetroundcap
\pgfsetroundjoin
\pgfsetmiterlimit{7.0}
\definecolor{eps2pgf_color}{gray}{0}\pgfsetstrokecolor{eps2pgf_color}\pgfsetfillcolor{eps2pgf_color}
\pgfpathmoveto{\pgfqpoint{0.117cm}{1.421cm}}
\pgfpathlineto{\pgfqpoint{0.682cm}{0.671cm}}
\pgfpathlineto{\pgfqpoint{1.246cm}{1.421cm}}
\pgfusepath{stroke}
\end{pgfscope}
\definecolor{eps2pgf_color}{gray}{0}\pgfsetstrokecolor{eps2pgf_color}\pgfsetfillcolor{eps2pgf_color}
\pgfpathmoveto{\pgfqpoint{0.273cm}{1.395cm}}
\pgfpathcurveto{\pgfqpoint{0.273cm}{1.432cm}}{\pgfqpoint{0.259cm}{1.467cm}}{\pgfqpoint{0.233cm}{1.492cm}}
\pgfpathcurveto{\pgfqpoint{0.207cm}{1.518cm}}{\pgfqpoint{0.173cm}{1.532cm}}{\pgfqpoint{0.137cm}{1.532cm}}
\pgfpathcurveto{\pgfqpoint{0.1cm}{1.532cm}}{\pgfqpoint{0.066cm}{1.518cm}}{\pgfqpoint{0.04cm}{1.492cm}}
\pgfpathcurveto{\pgfqpoint{0.014cm}{1.467cm}}{\pgfqpoint{0cm}{1.432cm}}{\pgfqpoint{0cm}{1.395cm}}
\pgfpathcurveto{\pgfqpoint{0cm}{1.359cm}}{\pgfqpoint{0.014cm}{1.324cm}}{\pgfqpoint{0.04cm}{1.299cm}}
\pgfpathcurveto{\pgfqpoint{0.066cm}{1.273cm}}{\pgfqpoint{0.1cm}{1.258cm}}{\pgfqpoint{0.137cm}{1.258cm}}
\pgfpathcurveto{\pgfqpoint{0.173cm}{1.258cm}}{\pgfqpoint{0.207cm}{1.273cm}}{\pgfqpoint{0.233cm}{1.299cm}}
\pgfpathcurveto{\pgfqpoint{0.259cm}{1.324cm}}{\pgfqpoint{0.273cm}{1.359cm}}{\pgfqpoint{0.273cm}{1.395cm}}
\pgfusepath{fill}
\begin{pgfscope}
\pgfsetdash{}{0cm}
\pgfsetlinewidth{0.818mm}
\pgfsetmiterlimit{7.0}
\pgfpathmoveto{\pgfqpoint{0.682cm}{0.671cm}}
\pgfpathlineto{\pgfqpoint{0.679cm}{1.418cm}}
\pgfusepath{stroke}
\end{pgfscope}
\pgfpathmoveto{\pgfqpoint{0.815cm}{1.399cm}}
\pgfpathcurveto{\pgfqpoint{0.815cm}{1.435cm}}{\pgfqpoint{0.801cm}{1.47cm}}{\pgfqpoint{0.775cm}{1.496cm}}
\pgfpathcurveto{\pgfqpoint{0.75cm}{1.521cm}}{\pgfqpoint{0.715cm}{1.536cm}}{\pgfqpoint{0.679cm}{1.536cm}}
\pgfpathcurveto{\pgfqpoint{0.643cm}{1.536cm}}{\pgfqpoint{0.608cm}{1.521cm}}{\pgfqpoint{0.582cm}{1.496cm}}
\pgfpathcurveto{\pgfqpoint{0.557cm}{1.47cm}}{\pgfqpoint{0.542cm}{1.435cm}}{\pgfqpoint{0.542cm}{1.399cm}}
\pgfpathcurveto{\pgfqpoint{0.542cm}{1.363cm}}{\pgfqpoint{0.557cm}{1.328cm}}{\pgfqpoint{0.582cm}{1.302cm}}
\pgfpathcurveto{\pgfqpoint{0.608cm}{1.276cm}}{\pgfqpoint{0.643cm}{1.262cm}}{\pgfqpoint{0.679cm}{1.262cm}}
\pgfpathcurveto{\pgfqpoint{0.715cm}{1.262cm}}{\pgfqpoint{0.75cm}{1.276cm}}{\pgfqpoint{0.775cm}{1.302cm}}
\pgfpathcurveto{\pgfqpoint{0.801cm}{1.328cm}}{\pgfqpoint{0.815cm}{1.363cm}}{\pgfqpoint{0.815cm}{1.399cm}}
\pgfusepath{fill}
\pgfpathmoveto{\pgfqpoint{1.345cm}{1.371cm}}
\pgfpathcurveto{\pgfqpoint{1.345cm}{1.408cm}}{\pgfqpoint{1.331cm}{1.442cm}}{\pgfqpoint{1.305cm}{1.468cm}}
\pgfpathcurveto{\pgfqpoint{1.28cm}{1.494cm}}{\pgfqpoint{1.245cm}{1.508cm}}{\pgfqpoint{1.209cm}{1.508cm}}
\pgfpathcurveto{\pgfqpoint{1.172cm}{1.508cm}}{\pgfqpoint{1.138cm}{1.494cm}}{\pgfqpoint{1.112cm}{1.468cm}}
\pgfpathcurveto{\pgfqpoint{1.087cm}{1.442cm}}{\pgfqpoint{1.072cm}{1.408cm}}{\pgfqpoint{1.072cm}{1.371cm}}
\pgfpathcurveto{\pgfqpoint{1.072cm}{1.335cm}}{\pgfqpoint{1.087cm}{1.3cm}}{\pgfqpoint{1.112cm}{1.274cm}}
\pgfpathcurveto{\pgfqpoint{1.138cm}{1.249cm}}{\pgfqpoint{1.172cm}{1.234cm}}{\pgfqpoint{1.209cm}{1.234cm}}
\pgfpathcurveto{\pgfqpoint{1.245cm}{1.234cm}}{\pgfqpoint{1.28cm}{1.249cm}}{\pgfqpoint{1.305cm}{1.274cm}}
\pgfpathcurveto{\pgfqpoint{1.331cm}{1.3cm}}{\pgfqpoint{1.345cm}{1.335cm}}{\pgfqpoint{1.345cm}{1.371cm}}
\pgfusepath{fill}
\begin{pgfscope}
\pgfsetdash{}{0cm}
\pgfsetlinewidth{0.818mm}
\pgfsetroundcap
\pgfsetmiterlimit{4.0}
\pgfpathmoveto{\pgfqpoint{0.682cm}{0.671cm}}
\pgfpathlineto{\pgfqpoint{0.682cm}{0.042cm}}
\pgfusepath{stroke}
\end{pgfscope}
\end{pgfscope}
\end{pgfscope}
\end{pgfscope}
\end{tikzpicture}}}, X_{\varepsilon}
     ) -\lambda  X_{\varepsilon} \circ \left( 2 Y_{\varepsilon} \prec
     \LL_{\varepsilon}^{- 1} \left[ 3 \left( \UU^{\varepsilon}_{>} \llbracket
     X_{\varepsilon}^2 \rrbracket \right) \succ Y_{\varepsilon} \right]
     \right) + X_{\varepsilon} \circ (Y_{\varepsilon} \circ Y_{\varepsilon}) .
  \]
  Accordingly,
  \[ \| X_{\varepsilon} \circ Y_{\varepsilon}^2 \|_{C_T \CC^{- \kappa,
     \varepsilon} (\rho^{4 \sigma})} \lesssim \lambda \|
     X_{\varepsilon}^{\!\resizebox{!}{.8em}{
\begin{tikzpicture}
\pgfpathmoveto{\pgfqpoint{0cm}{-0.035cm}}
\pgfpathlineto{\pgfqpoint{1.976cm}{-0.035cm}}
\pgfpathlineto{\pgfqpoint{1.976cm}{1.94cm}}
\pgfpathlineto{\pgfqpoint{0cm}{1.94cm}}
\pgfpathclose
\pgfusepath{clip}
\begin{pgfscope}
\begin{pgfscope}
\pgfpathmoveto{\pgfqpoint{0cm}{-0.035cm}}
\pgfpathlineto{\pgfqpoint{1.976cm}{-0.035cm}}
\pgfpathlineto{\pgfqpoint{1.976cm}{1.94cm}}
\pgfpathlineto{\pgfqpoint{0cm}{1.94cm}}
\pgfpathclose
\pgfusepath{clip}
\begin{pgfscope}
\begin{pgfscope}
\pgfsetdash{}{0cm}
\pgfsetlinewidth{0.818mm}
\pgfsetroundcap
\pgfsetroundjoin
\pgfsetmiterlimit{7.0}
\definecolor{eps2pgf_color}{gray}{0}\pgfsetstrokecolor{eps2pgf_color}\pgfsetfillcolor{eps2pgf_color}
\pgfpathmoveto{\pgfqpoint{0.117cm}{1.815cm}}
\pgfpathlineto{\pgfqpoint{0.682cm}{1.065cm}}
\pgfpathlineto{\pgfqpoint{1.246cm}{1.815cm}}
\pgfusepath{stroke}
\end{pgfscope}
\definecolor{eps2pgf_color}{gray}{0}\pgfsetstrokecolor{eps2pgf_color}\pgfsetfillcolor{eps2pgf_color}
\pgfpathmoveto{\pgfqpoint{0.273cm}{1.789cm}}
\pgfpathcurveto{\pgfqpoint{0.273cm}{1.825cm}}{\pgfqpoint{0.259cm}{1.86cm}}{\pgfqpoint{0.233cm}{1.886cm}}
\pgfpathcurveto{\pgfqpoint{0.207cm}{1.912cm}}{\pgfqpoint{0.173cm}{1.926cm}}{\pgfqpoint{0.137cm}{1.926cm}}
\pgfpathcurveto{\pgfqpoint{0.1cm}{1.926cm}}{\pgfqpoint{0.066cm}{1.912cm}}{\pgfqpoint{0.04cm}{1.886cm}}
\pgfpathcurveto{\pgfqpoint{0.014cm}{1.86cm}}{\pgfqpoint{0cm}{1.825cm}}{\pgfqpoint{0cm}{1.789cm}}
\pgfpathcurveto{\pgfqpoint{0cm}{1.753cm}}{\pgfqpoint{0.014cm}{1.718cm}}{\pgfqpoint{0.04cm}{1.692cm}}
\pgfpathcurveto{\pgfqpoint{0.066cm}{1.667cm}}{\pgfqpoint{0.1cm}{1.652cm}}{\pgfqpoint{0.137cm}{1.652cm}}
\pgfpathcurveto{\pgfqpoint{0.173cm}{1.652cm}}{\pgfqpoint{0.207cm}{1.667cm}}{\pgfqpoint{0.233cm}{1.692cm}}
\pgfpathcurveto{\pgfqpoint{0.259cm}{1.718cm}}{\pgfqpoint{0.273cm}{1.753cm}}{\pgfqpoint{0.273cm}{1.789cm}}
\pgfusepath{fill}
\begin{pgfscope}
\pgfsetdash{}{0cm}
\pgfsetlinewidth{0.818mm}
\pgfsetmiterlimit{7.0}
\pgfpathmoveto{\pgfqpoint{0.682cm}{1.065cm}}
\pgfpathlineto{\pgfqpoint{0.679cm}{1.812cm}}
\pgfusepath{stroke}
\end{pgfscope}
\pgfpathmoveto{\pgfqpoint{0.815cm}{1.793cm}}
\pgfpathcurveto{\pgfqpoint{0.815cm}{1.829cm}}{\pgfqpoint{0.801cm}{1.864cm}}{\pgfqpoint{0.775cm}{1.89cm}}
\pgfpathcurveto{\pgfqpoint{0.75cm}{1.915cm}}{\pgfqpoint{0.715cm}{1.93cm}}{\pgfqpoint{0.679cm}{1.93cm}}
\pgfpathcurveto{\pgfqpoint{0.643cm}{1.93cm}}{\pgfqpoint{0.608cm}{1.915cm}}{\pgfqpoint{0.582cm}{1.89cm}}
\pgfpathcurveto{\pgfqpoint{0.557cm}{1.864cm}}{\pgfqpoint{0.542cm}{1.829cm}}{\pgfqpoint{0.542cm}{1.793cm}}
\pgfpathcurveto{\pgfqpoint{0.542cm}{1.756cm}}{\pgfqpoint{0.557cm}{1.722cm}}{\pgfqpoint{0.582cm}{1.696cm}}
\pgfpathcurveto{\pgfqpoint{0.608cm}{1.67cm}}{\pgfqpoint{0.643cm}{1.656cm}}{\pgfqpoint{0.679cm}{1.656cm}}
\pgfpathcurveto{\pgfqpoint{0.715cm}{1.656cm}}{\pgfqpoint{0.75cm}{1.67cm}}{\pgfqpoint{0.775cm}{1.696cm}}
\pgfpathcurveto{\pgfqpoint{0.801cm}{1.722cm}}{\pgfqpoint{0.815cm}{1.756cm}}{\pgfqpoint{0.815cm}{1.793cm}}
\pgfusepath{fill}
\pgfpathmoveto{\pgfqpoint{1.345cm}{1.765cm}}
\pgfpathcurveto{\pgfqpoint{1.345cm}{1.801cm}}{\pgfqpoint{1.331cm}{1.836cm}}{\pgfqpoint{1.305cm}{1.862cm}}
\pgfpathcurveto{\pgfqpoint{1.28cm}{1.887cm}}{\pgfqpoint{1.245cm}{1.902cm}}{\pgfqpoint{1.209cm}{1.902cm}}
\pgfpathcurveto{\pgfqpoint{1.172cm}{1.902cm}}{\pgfqpoint{1.138cm}{1.887cm}}{\pgfqpoint{1.112cm}{1.862cm}}
\pgfpathcurveto{\pgfqpoint{1.087cm}{1.836cm}}{\pgfqpoint{1.072cm}{1.801cm}}{\pgfqpoint{1.072cm}{1.765cm}}
\pgfpathcurveto{\pgfqpoint{1.072cm}{1.728cm}}{\pgfqpoint{1.087cm}{1.694cm}}{\pgfqpoint{1.112cm}{1.668cm}}
\pgfpathcurveto{\pgfqpoint{1.138cm}{1.642cm}}{\pgfqpoint{1.172cm}{1.628cm}}{\pgfqpoint{1.209cm}{1.628cm}}
\pgfpathcurveto{\pgfqpoint{1.245cm}{1.628cm}}{\pgfqpoint{1.28cm}{1.642cm}}{\pgfqpoint{1.305cm}{1.668cm}}
\pgfpathcurveto{\pgfqpoint{1.331cm}{1.694cm}}{\pgfqpoint{1.345cm}{1.728cm}}{\pgfqpoint{1.345cm}{1.765cm}}
\pgfusepath{fill}
\begin{pgfscope}
\pgfsetdash{}{0cm}
\pgfsetlinewidth{0.818mm}
\pgfsetroundcap
\pgfsetroundjoin
\pgfsetmiterlimit{7.0}
\pgfpathmoveto{\pgfqpoint{0.682cm}{1.065cm}}
\pgfpathlineto{\pgfqpoint{1.246cm}{0.315cm}}
\pgfpathlineto{\pgfqpoint{1.811cm}{1.065cm}}
\pgfusepath{stroke}
\end{pgfscope}
\pgfpathmoveto{\pgfqpoint{1.948cm}{1.065cm}}
\pgfpathcurveto{\pgfqpoint{1.948cm}{1.101cm}}{\pgfqpoint{1.933cm}{1.136cm}}{\pgfqpoint{1.907cm}{1.162cm}}
\pgfpathcurveto{\pgfqpoint{1.882cm}{1.187cm}}{\pgfqpoint{1.847cm}{1.202cm}}{\pgfqpoint{1.811cm}{1.202cm}}
\pgfpathcurveto{\pgfqpoint{1.775cm}{1.202cm}}{\pgfqpoint{1.74cm}{1.187cm}}{\pgfqpoint{1.714cm}{1.162cm}}
\pgfpathcurveto{\pgfqpoint{1.689cm}{1.136cm}}{\pgfqpoint{1.674cm}{1.101cm}}{\pgfqpoint{1.674cm}{1.065cm}}
\pgfpathcurveto{\pgfqpoint{1.674cm}{1.029cm}}{\pgfqpoint{1.689cm}{0.994cm}}{\pgfqpoint{1.714cm}{0.968cm}}
\pgfpathcurveto{\pgfqpoint{1.74cm}{0.942cm}}{\pgfqpoint{1.775cm}{0.928cm}}{\pgfqpoint{1.811cm}{0.928cm}}
\pgfpathcurveto{\pgfqpoint{1.847cm}{0.928cm}}{\pgfqpoint{1.882cm}{0.942cm}}{\pgfqpoint{1.907cm}{0.968cm}}
\pgfpathcurveto{\pgfqpoint{1.933cm}{0.994cm}}{\pgfqpoint{1.948cm}{1.029cm}}{\pgfqpoint{1.948cm}{1.065cm}}
\pgfusepath{fill}
\begin{pgfscope}
\pgfsetdash{}{0cm}
\pgfsetlinewidth{0.818mm}
\pgfsetmiterlimit{4.0}
\pgfpathmoveto{\pgfqpoint{1.383cm}{0.178cm}}
\pgfpathcurveto{\pgfqpoint{1.383cm}{0.214cm}}{\pgfqpoint{1.369cm}{0.249cm}}{\pgfqpoint{1.343cm}{0.275cm}}
\pgfpathcurveto{\pgfqpoint{1.317cm}{0.3cm}}{\pgfqpoint{1.283cm}{0.315cm}}{\pgfqpoint{1.246cm}{0.315cm}}
\pgfpathcurveto{\pgfqpoint{1.21cm}{0.315cm}}{\pgfqpoint{1.175cm}{0.3cm}}{\pgfqpoint{1.15cm}{0.275cm}}
\pgfpathcurveto{\pgfqpoint{1.124cm}{0.249cm}}{\pgfqpoint{1.11cm}{0.214cm}}{\pgfqpoint{1.11cm}{0.178cm}}
\pgfpathcurveto{\pgfqpoint{1.11cm}{0.141cm}}{\pgfqpoint{1.124cm}{0.107cm}}{\pgfqpoint{1.15cm}{0.081cm}}
\pgfpathcurveto{\pgfqpoint{1.175cm}{0.055cm}}{\pgfqpoint{1.21cm}{0.041cm}}{\pgfqpoint{1.246cm}{0.041cm}}
\pgfpathcurveto{\pgfqpoint{1.283cm}{0.041cm}}{\pgfqpoint{1.317cm}{0.055cm}}{\pgfqpoint{1.343cm}{0.081cm}}
\pgfpathcurveto{\pgfqpoint{1.369cm}{0.107cm}}{\pgfqpoint{1.383cm}{0.141cm}}{\pgfqpoint{1.383cm}{0.178cm}}
\pgfusepath{stroke}
\end{pgfscope}
\end{pgfscope}
\end{pgfscope}
\end{pgfscope}
\end{tikzpicture}}} \|_{C_T \CC^{- \kappa, \varepsilon}
     (\rho^{\sigma})} \| Y_{\varepsilon} \|_{C_T \CC^{2 \kappa,
     \varepsilon} (\rho^{\sigma})} \]
  \[ + \lambda \| Y_{\varepsilon} \|_{C \CC^{3\kappa, \varepsilon}
     (\rho^{\sigma})} \| X_{\varepsilon}^{\!\resizebox{0.6em}{!}{
\begin{tikzpicture}
\pgfpathmoveto{\pgfqpoint{0cm}{-0.035cm}}
\pgfpathlineto{\pgfqpoint{1.376cm}{-0.035cm}}
\pgfpathlineto{\pgfqpoint{1.376cm}{1.552cm}}
\pgfpathlineto{\pgfqpoint{0cm}{1.552cm}}
\pgfpathclose
\pgfusepath{clip}
\begin{pgfscope}
\begin{pgfscope}
\pgfpathmoveto{\pgfqpoint{0cm}{-0.035cm}}
\pgfpathlineto{\pgfqpoint{1.376cm}{-0.035cm}}
\pgfpathlineto{\pgfqpoint{1.376cm}{1.552cm}}
\pgfpathlineto{\pgfqpoint{0cm}{1.552cm}}
\pgfpathclose
\pgfusepath{clip}
\begin{pgfscope}
\begin{pgfscope}
\pgfsetdash{}{0cm}
\pgfsetlinewidth{0.818mm}
\pgfsetroundcap
\pgfsetroundjoin
\pgfsetmiterlimit{7.0}
\definecolor{eps2pgf_color}{gray}{0}\pgfsetstrokecolor{eps2pgf_color}\pgfsetfillcolor{eps2pgf_color}
\pgfpathmoveto{\pgfqpoint{0.117cm}{1.421cm}}
\pgfpathlineto{\pgfqpoint{0.682cm}{0.671cm}}
\pgfpathlineto{\pgfqpoint{1.246cm}{1.421cm}}
\pgfusepath{stroke}
\end{pgfscope}
\definecolor{eps2pgf_color}{gray}{0}\pgfsetstrokecolor{eps2pgf_color}\pgfsetfillcolor{eps2pgf_color}
\pgfpathmoveto{\pgfqpoint{0.273cm}{1.395cm}}
\pgfpathcurveto{\pgfqpoint{0.273cm}{1.432cm}}{\pgfqpoint{0.259cm}{1.467cm}}{\pgfqpoint{0.233cm}{1.492cm}}
\pgfpathcurveto{\pgfqpoint{0.207cm}{1.518cm}}{\pgfqpoint{0.173cm}{1.532cm}}{\pgfqpoint{0.137cm}{1.532cm}}
\pgfpathcurveto{\pgfqpoint{0.1cm}{1.532cm}}{\pgfqpoint{0.066cm}{1.518cm}}{\pgfqpoint{0.04cm}{1.492cm}}
\pgfpathcurveto{\pgfqpoint{0.014cm}{1.467cm}}{\pgfqpoint{0cm}{1.432cm}}{\pgfqpoint{0cm}{1.395cm}}
\pgfpathcurveto{\pgfqpoint{0cm}{1.359cm}}{\pgfqpoint{0.014cm}{1.324cm}}{\pgfqpoint{0.04cm}{1.299cm}}
\pgfpathcurveto{\pgfqpoint{0.066cm}{1.273cm}}{\pgfqpoint{0.1cm}{1.258cm}}{\pgfqpoint{0.137cm}{1.258cm}}
\pgfpathcurveto{\pgfqpoint{0.173cm}{1.258cm}}{\pgfqpoint{0.207cm}{1.273cm}}{\pgfqpoint{0.233cm}{1.299cm}}
\pgfpathcurveto{\pgfqpoint{0.259cm}{1.324cm}}{\pgfqpoint{0.273cm}{1.359cm}}{\pgfqpoint{0.273cm}{1.395cm}}
\pgfusepath{fill}
\begin{pgfscope}
\pgfsetdash{}{0cm}
\pgfsetlinewidth{0.818mm}
\pgfsetmiterlimit{7.0}
\pgfpathmoveto{\pgfqpoint{0.682cm}{0.671cm}}
\pgfpathlineto{\pgfqpoint{0.679cm}{1.418cm}}
\pgfusepath{stroke}
\end{pgfscope}
\pgfpathmoveto{\pgfqpoint{0.815cm}{1.399cm}}
\pgfpathcurveto{\pgfqpoint{0.815cm}{1.435cm}}{\pgfqpoint{0.801cm}{1.47cm}}{\pgfqpoint{0.775cm}{1.496cm}}
\pgfpathcurveto{\pgfqpoint{0.75cm}{1.521cm}}{\pgfqpoint{0.715cm}{1.536cm}}{\pgfqpoint{0.679cm}{1.536cm}}
\pgfpathcurveto{\pgfqpoint{0.643cm}{1.536cm}}{\pgfqpoint{0.608cm}{1.521cm}}{\pgfqpoint{0.582cm}{1.496cm}}
\pgfpathcurveto{\pgfqpoint{0.557cm}{1.47cm}}{\pgfqpoint{0.542cm}{1.435cm}}{\pgfqpoint{0.542cm}{1.399cm}}
\pgfpathcurveto{\pgfqpoint{0.542cm}{1.363cm}}{\pgfqpoint{0.557cm}{1.328cm}}{\pgfqpoint{0.582cm}{1.302cm}}
\pgfpathcurveto{\pgfqpoint{0.608cm}{1.276cm}}{\pgfqpoint{0.643cm}{1.262cm}}{\pgfqpoint{0.679cm}{1.262cm}}
\pgfpathcurveto{\pgfqpoint{0.715cm}{1.262cm}}{\pgfqpoint{0.75cm}{1.276cm}}{\pgfqpoint{0.775cm}{1.302cm}}
\pgfpathcurveto{\pgfqpoint{0.801cm}{1.328cm}}{\pgfqpoint{0.815cm}{1.363cm}}{\pgfqpoint{0.815cm}{1.399cm}}
\pgfusepath{fill}
\pgfpathmoveto{\pgfqpoint{1.345cm}{1.371cm}}
\pgfpathcurveto{\pgfqpoint{1.345cm}{1.408cm}}{\pgfqpoint{1.331cm}{1.442cm}}{\pgfqpoint{1.305cm}{1.468cm}}
\pgfpathcurveto{\pgfqpoint{1.28cm}{1.494cm}}{\pgfqpoint{1.245cm}{1.508cm}}{\pgfqpoint{1.209cm}{1.508cm}}
\pgfpathcurveto{\pgfqpoint{1.172cm}{1.508cm}}{\pgfqpoint{1.138cm}{1.494cm}}{\pgfqpoint{1.112cm}{1.468cm}}
\pgfpathcurveto{\pgfqpoint{1.087cm}{1.442cm}}{\pgfqpoint{1.072cm}{1.408cm}}{\pgfqpoint{1.072cm}{1.371cm}}
\pgfpathcurveto{\pgfqpoint{1.072cm}{1.335cm}}{\pgfqpoint{1.087cm}{1.3cm}}{\pgfqpoint{1.112cm}{1.274cm}}
\pgfpathcurveto{\pgfqpoint{1.138cm}{1.249cm}}{\pgfqpoint{1.172cm}{1.234cm}}{\pgfqpoint{1.209cm}{1.234cm}}
\pgfpathcurveto{\pgfqpoint{1.245cm}{1.234cm}}{\pgfqpoint{1.28cm}{1.249cm}}{\pgfqpoint{1.305cm}{1.274cm}}
\pgfpathcurveto{\pgfqpoint{1.331cm}{1.3cm}}{\pgfqpoint{1.345cm}{1.335cm}}{\pgfqpoint{1.345cm}{1.371cm}}
\pgfusepath{fill}
\begin{pgfscope}
\pgfsetdash{}{0cm}
\pgfsetlinewidth{0.818mm}
\pgfsetroundcap
\pgfsetmiterlimit{4.0}
\pgfpathmoveto{\pgfqpoint{0.682cm}{0.671cm}}
\pgfpathlineto{\pgfqpoint{0.682cm}{0.042cm}}
\pgfusepath{stroke}
\end{pgfscope}
\end{pgfscope}
\end{pgfscope}
\end{pgfscope}
\end{tikzpicture}}} \|_{C_T \CC^{1
     / 2 - \kappa, \varepsilon} (\rho^{\sigma})} \| X_{\varepsilon} \|_{C_T
     \CC^{- 1 / 2 - \kappa, \varepsilon} (\rho^{\sigma})} \]
  \[ + \lambda \| X_{\varepsilon} \|_{C_T \CC^{- 1 / 2 - \kappa, \varepsilon}
     (\rho^{\sigma})} \| Y_{\varepsilon} \|_{C_T L^{\infty,
     \varepsilon} (\rho^{\sigma})}^2 \| \llbracket X_{\varepsilon}^2
     \rrbracket \|_{C_T \CC^{- 1 - \kappa, \varepsilon} (\rho^{\sigma})} \]
  \begin{equation}\label{eq:XY2-res}
  + \| X_{\varepsilon} \|_{C_T \CC^{- 1 / 2 - \kappa, \varepsilon}
     (\rho^{\sigma})} \| Y_{\varepsilon} \|_{C_T \CC^{3\kappa,
     \varepsilon} (\rho^{\sigma})}  \| Y_{\varepsilon} \|_{C_T \CC^{1 / 2 - \kappa,
     \varepsilon} (\rho^{\sigma})} \lesssim (\lambda^2 + \lambda^3) 
\|\mathbb{X}_{\varepsilon}\|^{9} 
  \end{equation}
  and for the paraproducts
  \[ \| X_{\varepsilon} \prec Y_{\varepsilon}^2 \|_{C_T \CC^{- 2 \kappa,
     \varepsilon} (\rho^{4 \sigma})} \lesssim \| X_{\varepsilon} \|_{C_T
     \CC^{- 1 / 2 - \kappa, \varepsilon} (\rho^{\sigma})} \| Y_{\varepsilon}
     \|^2_{C_T \CC^{1 / 2 - \kappa, \varepsilon} (\rho^{\sigma})} \lesssim \lambda^2
     \|\mathbb{X}_{\varepsilon}\|^{7}, \]
  \[ \| X_{\varepsilon} \succ Y_{\varepsilon}^2 \|_{C_T \CC^{- 1 / 2 - \kappa,
     \varepsilon} (\rho^{4 \sigma})} \lesssim \| X_{\varepsilon} \|_{C_T
     \CC^{- 1 / 2 - \kappa, \varepsilon} (\rho^{\sigma})} \| Y_{\varepsilon}
     \|^2_{C_T L^{\infty, \varepsilon} (\rho^{\sigma})} \lesssim \lambda^2 
     \|\mathbb{X}_{\varepsilon}\|^{7} . \]
       This gives the second bound from the statement of the lemma.
\end{proof}

Let us now proceed with our main energy estimate. In view of Lemma
\ref{lem:energy12}, our goal is to control the terms in $\Theta_{\rho^4,
\varepsilon} + \Psi_{\rho^4, \varepsilon}$ by quantities of the from
\[ c(\lambda) Q_{\rho} (\mathbb{X}_{\varepsilon}) + \delta (\lambda \| \rho \phi_{\varepsilon}
   \|_{L^{4, \varepsilon}}^4 + m^{2} \| \rho^2 \psi_{\varepsilon} \|_{L^{2,
   \varepsilon}}^2 + \| \rho^2 \nabla_{\varepsilon} \psi_{\varepsilon}
   \|_{L^{2, \varepsilon}}^2), \]
where $\delta > 0$ is a small constant which can change from line to line.
Indeed, with such a bound in hand it will be possible to absorb the norms of
$\phi_{\varepsilon}, \psi_{\varepsilon}$ from the right hand side of
{\eqref{eq:en12}} into the left hand side and a bound for $\phi_{\varepsilon},
\psi_{\varepsilon}$ in terms of the noise terms will follow.

\begin{lemma}
  \label{lemma:bounds-rhs1}Let $\rho$ be a weight such that $\rho^{\iota} \in
  L^{4, 0}$ for some $\iota \in (0, 1)$. Then
  \[ | \Theta_{\rho^4, \varepsilon} | + | \Psi_{\rho^4, \varepsilon} |
     \leqslant ( \lambda^3+\lambda^{(12 - \theta)
     / (2 + \theta)} | \log t|^{4 / (2 + \theta)}+ \lambda^{7}) Q_{\rho} (\mathbb{X}_{\varepsilon}) \]
  \[ + \delta (\lambda \| \rho \phi_{\varepsilon} \|_{L^{4, \varepsilon}}^4 + \|
     \rho^2 \phi_{\varepsilon} \|_{H^{1 - 2 \kappa, \varepsilon}}^2 + m^{2}\|
     \rho^2 \psi_{\varepsilon} \|^2_{L^{2, \varepsilon}} + \| \rho^2 \nabla_{\varepsilon}
     \psi_{\varepsilon} \|^2_{L^{2, \varepsilon}}), \]
where $\theta=\frac{1/2-4\kappa}{1-2\kappa}$.
\end{lemma}

\begin{proof}
  Since the weight $\rho$ is polynomial and vanishes at infinity, we may
  assume without loss of generality that $0 < \rho \leqslant 1$ and
  consequently $\rho^{\alpha} \leqslant \rho^{\beta}$ whenever $\alpha
  \geqslant \beta \geqslant 0$. We also observe that due to the integrability
  of the weight (see Lemma~\ref{lem:15})
  \[ \| \rho^{1 + \iota} \phi_{\varepsilon} \|_{L^{2, \varepsilon}} \lesssim
     \| \rho \phi_{\varepsilon} \|_{L^{4, \varepsilon}} \]
  with a constant that depends only on $\rho$. In the sequel, we repeatedly
  use various results for discrete Besov spaces established in
  Section~\ref{s:app}. Namely, the equivalent formulation of the Besov norms
  (Lemma~\ref{lem:equiv2}), the duality estimate (Lemma~\ref{lem:dual2}),
  interpolation (Lemma~\ref{lem:int}), embedding (Lemma~\ref{lem:emb}), a
  bound for powers of functions (Lemma~\ref{lem:mult}) as well as bounds for
  the commutators (Lemma~\ref{lem:comm1}).
  
  Even though it is not necessary for the present proof, we keep track of the precise power of the quantity $\|\mathbb {X}_{\varepsilon}\|$ in each of the estimates. This will be used in Section \ref{s:exp} below to establish the stretched exponential integrability of the fields. We recall that $\vartheta=O(\kappa)>0$ denotes a generic small constant which changes from line to line.
  
  In view of Lemma~\ref{lem:energy12} we shall bound each term on the right
  hand side of {\eqref{eq:en12}}. We have
  \[ | \langle [\nabla_{\varepsilon}, \rho^4] \psi_{\varepsilon},
     \nabla_{\varepsilon} \psi_{\varepsilon} \rangle_{\varepsilon} | \leqslant
     C_{\rho}  \| \rho^2 \psi_{\varepsilon} \|_{L^{2, \varepsilon}}  \| \rho^2
     \nabla_{\varepsilon} \psi_{\varepsilon} \|_{L^{2, \varepsilon}} \leqslant
     C_{\delta} C_{\rho}^2  \| \rho^2 \psi_{\varepsilon} \|^2_{L^{2,
     \varepsilon}} + \delta \| \rho^2 \nabla_{\varepsilon} \psi_{\varepsilon}
     \|^2_{L^{2, \varepsilon}} . \]
  This term can be absorbed provided $C_{\rho} = \| \rho^{- 4}
  [\nabla_{\varepsilon}, \rho^4]\|_{L^{\infty, \varepsilon}}$ is sufficiently
  small, such that $C_{\delta} C^2_{\rho} \leqslant m^2$, which can be
  obtained by choosing $h > 0$ small enough (depending only on $m^2$ and
  $\delta$) in the definition~{\eqref{eq:weight}} of the weight $\rho$. Next,
  \[ \left| \left\langle \left[ \Q_{\varepsilon}, \rho^4 \right]
     \Q_{\varepsilon}^{- 1} [3 \lambda \llbracket X_{\varepsilon}^2 \rrbracket
     \succ \phi_{\varepsilon}], \psi_{\varepsilon} \right\rangle_{\varepsilon}
     \right| \leqslant \left| \left\langle \Q_{\varepsilon}^{- 1} [3 \lambda
     \llbracket X_{\varepsilon}^2 \rrbracket \succ \phi_{\varepsilon}], \left[
     \Q_{\varepsilon}, \rho^4 \right] \psi_{\varepsilon}
     \right\rangle_{\varepsilon} \right| \]
  and we estimate explicitly
  \[ \left| \rho^{- 2} \left[ \Q_{\varepsilon}, \rho^4 \right]
     \psi_{\varepsilon} \right|_{L^{2, \varepsilon}} \leqslant C_{\rho}  (\|
     \rho^2 \psi_{\varepsilon} \|_{L^{2, \varepsilon}} +\| \rho^2
     \nabla_{\varepsilon} \psi_{\varepsilon} \|_{L^{2, \varepsilon} (\rho^2)})
  \]
  for another constant $C_{\rho}$ depending only on the weight $\rho$, which
  can be taken smaller than $m^2$ by choosing $h > 0$ small, and consequently
  \[ \left| \left\langle \left[ \Q_{\varepsilon}, \rho^4 \right]
     \Q_{\varepsilon}^{- 1} [3 \lambda \llbracket X_{\varepsilon}^2 \rrbracket
     \succ \phi_{\varepsilon}], \psi_{\varepsilon} \right\rangle_{\varepsilon}
     \right| \lesssim \lambda \|\mathbb{X}_{\varepsilon} \|^2  \| \rho^{2 -
     \sigma} \phi_{\varepsilon} \|_{L^{2, \varepsilon}}  (m^2 \| \rho^2
     \psi_{\varepsilon} \|_{L^{2, \varepsilon}} +\| \rho^2
     \nabla_{\varepsilon} \psi_{\varepsilon} \|_{L^{2, \varepsilon}}) \]
  \[ \leqslant \lambda^3 C_{\delta} \|\mathbb{X}_{\varepsilon} \|^8 + \delta
     (\lambda \| \rho \phi_{\varepsilon} \|^4_{L^{4, \varepsilon}} + m^2 \|
     \rho^2 \psi_{\varepsilon} \|^2_{L^{2, \varepsilon}} +\| \rho^2
     \nabla_{\varepsilon} \psi_{\varepsilon} \|^2_{L^{2, \varepsilon}}), \]
  since $\sigma$ is sufficiently small.
  
  Using Lemma~\ref{lem:dual2}, Lemma~\ref{lem:mult}, interpolation from Lemma~\ref{lem:int} with for $\theta = \frac{1 -
  4 \kappa}{1 - 2 \kappa}$ and Young's inequality 
  we obtain  
  \[ | \lambda^2 \langle \rho^4 \phi_{\varepsilon}^2,
     X_{\varepsilon}^{\!\resizebox{!}{.8em}{
\begin{tikzpicture}
\pgfpathmoveto{\pgfqpoint{0cm}{-0.035cm}}
\pgfpathlineto{\pgfqpoint{1.976cm}{-0.035cm}}
\pgfpathlineto{\pgfqpoint{1.976cm}{1.94cm}}
\pgfpathlineto{\pgfqpoint{0cm}{1.94cm}}
\pgfpathclose
\pgfusepath{clip}
\begin{pgfscope}
\begin{pgfscope}
\pgfpathmoveto{\pgfqpoint{0cm}{-0.035cm}}
\pgfpathlineto{\pgfqpoint{1.976cm}{-0.035cm}}
\pgfpathlineto{\pgfqpoint{1.976cm}{1.94cm}}
\pgfpathlineto{\pgfqpoint{0cm}{1.94cm}}
\pgfpathclose
\pgfusepath{clip}
\begin{pgfscope}
\begin{pgfscope}
\pgfsetdash{}{0cm}
\pgfsetlinewidth{0.818mm}
\pgfsetroundcap
\pgfsetroundjoin
\pgfsetmiterlimit{7.0}
\definecolor{eps2pgf_color}{gray}{0}\pgfsetstrokecolor{eps2pgf_color}\pgfsetfillcolor{eps2pgf_color}
\pgfpathmoveto{\pgfqpoint{0.117cm}{1.815cm}}
\pgfpathlineto{\pgfqpoint{0.682cm}{1.065cm}}
\pgfpathlineto{\pgfqpoint{1.246cm}{1.815cm}}
\pgfusepath{stroke}
\end{pgfscope}
\definecolor{eps2pgf_color}{gray}{0}\pgfsetstrokecolor{eps2pgf_color}\pgfsetfillcolor{eps2pgf_color}
\pgfpathmoveto{\pgfqpoint{0.273cm}{1.789cm}}
\pgfpathcurveto{\pgfqpoint{0.273cm}{1.825cm}}{\pgfqpoint{0.259cm}{1.86cm}}{\pgfqpoint{0.233cm}{1.886cm}}
\pgfpathcurveto{\pgfqpoint{0.207cm}{1.912cm}}{\pgfqpoint{0.173cm}{1.926cm}}{\pgfqpoint{0.137cm}{1.926cm}}
\pgfpathcurveto{\pgfqpoint{0.1cm}{1.926cm}}{\pgfqpoint{0.066cm}{1.912cm}}{\pgfqpoint{0.04cm}{1.886cm}}
\pgfpathcurveto{\pgfqpoint{0.014cm}{1.86cm}}{\pgfqpoint{0cm}{1.825cm}}{\pgfqpoint{0cm}{1.789cm}}
\pgfpathcurveto{\pgfqpoint{0cm}{1.753cm}}{\pgfqpoint{0.014cm}{1.718cm}}{\pgfqpoint{0.04cm}{1.692cm}}
\pgfpathcurveto{\pgfqpoint{0.066cm}{1.667cm}}{\pgfqpoint{0.1cm}{1.652cm}}{\pgfqpoint{0.137cm}{1.652cm}}
\pgfpathcurveto{\pgfqpoint{0.173cm}{1.652cm}}{\pgfqpoint{0.207cm}{1.667cm}}{\pgfqpoint{0.233cm}{1.692cm}}
\pgfpathcurveto{\pgfqpoint{0.259cm}{1.718cm}}{\pgfqpoint{0.273cm}{1.753cm}}{\pgfqpoint{0.273cm}{1.789cm}}
\pgfusepath{fill}
\pgfpathmoveto{\pgfqpoint{1.345cm}{1.765cm}}
\pgfpathcurveto{\pgfqpoint{1.345cm}{1.801cm}}{\pgfqpoint{1.331cm}{1.836cm}}{\pgfqpoint{1.305cm}{1.862cm}}
\pgfpathcurveto{\pgfqpoint{1.28cm}{1.887cm}}{\pgfqpoint{1.245cm}{1.902cm}}{\pgfqpoint{1.209cm}{1.902cm}}
\pgfpathcurveto{\pgfqpoint{1.172cm}{1.902cm}}{\pgfqpoint{1.138cm}{1.887cm}}{\pgfqpoint{1.112cm}{1.862cm}}
\pgfpathcurveto{\pgfqpoint{1.087cm}{1.836cm}}{\pgfqpoint{1.072cm}{1.801cm}}{\pgfqpoint{1.072cm}{1.765cm}}
\pgfpathcurveto{\pgfqpoint{1.072cm}{1.728cm}}{\pgfqpoint{1.087cm}{1.694cm}}{\pgfqpoint{1.112cm}{1.668cm}}
\pgfpathcurveto{\pgfqpoint{1.138cm}{1.642cm}}{\pgfqpoint{1.172cm}{1.628cm}}{\pgfqpoint{1.209cm}{1.628cm}}
\pgfpathcurveto{\pgfqpoint{1.245cm}{1.628cm}}{\pgfqpoint{1.28cm}{1.642cm}}{\pgfqpoint{1.305cm}{1.668cm}}
\pgfpathcurveto{\pgfqpoint{1.331cm}{1.694cm}}{\pgfqpoint{1.345cm}{1.728cm}}{\pgfqpoint{1.345cm}{1.765cm}}
\pgfusepath{fill}
\begin{pgfscope}
\pgfsetdash{}{0cm}
\pgfsetlinewidth{0.818mm}
\pgfsetroundcap
\pgfsetroundjoin
\pgfsetmiterlimit{7.0}
\pgfpathmoveto{\pgfqpoint{0.682cm}{1.065cm}}
\pgfpathlineto{\pgfqpoint{1.246cm}{0.315cm}}
\pgfpathlineto{\pgfqpoint{1.811cm}{1.065cm}}
\pgfusepath{stroke}
\end{pgfscope}
\pgfpathmoveto{\pgfqpoint{1.948cm}{1.065cm}}
\pgfpathcurveto{\pgfqpoint{1.948cm}{1.101cm}}{\pgfqpoint{1.933cm}{1.136cm}}{\pgfqpoint{1.907cm}{1.162cm}}
\pgfpathcurveto{\pgfqpoint{1.882cm}{1.187cm}}{\pgfqpoint{1.847cm}{1.202cm}}{\pgfqpoint{1.811cm}{1.202cm}}
\pgfpathcurveto{\pgfqpoint{1.775cm}{1.202cm}}{\pgfqpoint{1.74cm}{1.187cm}}{\pgfqpoint{1.714cm}{1.162cm}}
\pgfpathcurveto{\pgfqpoint{1.689cm}{1.136cm}}{\pgfqpoint{1.674cm}{1.101cm}}{\pgfqpoint{1.674cm}{1.065cm}}
\pgfpathcurveto{\pgfqpoint{1.674cm}{1.029cm}}{\pgfqpoint{1.689cm}{0.994cm}}{\pgfqpoint{1.714cm}{0.968cm}}
\pgfpathcurveto{\pgfqpoint{1.74cm}{0.942cm}}{\pgfqpoint{1.775cm}{0.928cm}}{\pgfqpoint{1.811cm}{0.928cm}}
\pgfpathcurveto{\pgfqpoint{1.847cm}{0.928cm}}{\pgfqpoint{1.882cm}{0.942cm}}{\pgfqpoint{1.907cm}{0.968cm}}
\pgfpathcurveto{\pgfqpoint{1.933cm}{0.994cm}}{\pgfqpoint{1.948cm}{1.029cm}}{\pgfqpoint{1.948cm}{1.065cm}}
\pgfusepath{fill}
\begin{pgfscope}
\pgfsetdash{}{0cm}
\pgfsetlinewidth{0.818mm}
\pgfsetmiterlimit{7.0}
\pgfpathmoveto{\pgfqpoint{1.246cm}{0.315cm}}
\pgfpathlineto{\pgfqpoint{1.244cm}{1.061cm}}
\pgfusepath{stroke}
\end{pgfscope}
\pgfpathmoveto{\pgfqpoint{1.38cm}{1.065cm}}
\pgfpathcurveto{\pgfqpoint{1.38cm}{1.101cm}}{\pgfqpoint{1.366cm}{1.136cm}}{\pgfqpoint{1.34cm}{1.162cm}}
\pgfpathcurveto{\pgfqpoint{1.315cm}{1.187cm}}{\pgfqpoint{1.28cm}{1.202cm}}{\pgfqpoint{1.244cm}{1.202cm}}
\pgfpathcurveto{\pgfqpoint{1.207cm}{1.202cm}}{\pgfqpoint{1.173cm}{1.187cm}}{\pgfqpoint{1.147cm}{1.162cm}}
\pgfpathcurveto{\pgfqpoint{1.121cm}{1.136cm}}{\pgfqpoint{1.107cm}{1.101cm}}{\pgfqpoint{1.107cm}{1.065cm}}
\pgfpathcurveto{\pgfqpoint{1.107cm}{1.029cm}}{\pgfqpoint{1.121cm}{0.994cm}}{\pgfqpoint{1.147cm}{0.968cm}}
\pgfpathcurveto{\pgfqpoint{1.173cm}{0.942cm}}{\pgfqpoint{1.207cm}{0.928cm}}{\pgfqpoint{1.244cm}{0.928cm}}
\pgfpathcurveto{\pgfqpoint{1.28cm}{0.928cm}}{\pgfqpoint{1.315cm}{0.942cm}}{\pgfqpoint{1.34cm}{0.968cm}}
\pgfpathcurveto{\pgfqpoint{1.366cm}{0.994cm}}{\pgfqpoint{1.38cm}{1.029cm}}{\pgfqpoint{1.38cm}{1.065cm}}
\pgfusepath{fill}
\begin{pgfscope}
\pgfsetdash{}{0cm}
\pgfsetlinewidth{0.818mm}
\pgfsetmiterlimit{4.0}
\pgfpathmoveto{\pgfqpoint{1.383cm}{0.178cm}}
\pgfpathcurveto{\pgfqpoint{1.383cm}{0.214cm}}{\pgfqpoint{1.369cm}{0.249cm}}{\pgfqpoint{1.343cm}{0.275cm}}
\pgfpathcurveto{\pgfqpoint{1.317cm}{0.3cm}}{\pgfqpoint{1.283cm}{0.315cm}}{\pgfqpoint{1.246cm}{0.315cm}}
\pgfpathcurveto{\pgfqpoint{1.21cm}{0.315cm}}{\pgfqpoint{1.175cm}{0.3cm}}{\pgfqpoint{1.15cm}{0.275cm}}
\pgfpathcurveto{\pgfqpoint{1.124cm}{0.249cm}}{\pgfqpoint{1.11cm}{0.214cm}}{\pgfqpoint{1.11cm}{0.178cm}}
\pgfpathcurveto{\pgfqpoint{1.11cm}{0.141cm}}{\pgfqpoint{1.124cm}{0.107cm}}{\pgfqpoint{1.15cm}{0.081cm}}
\pgfpathcurveto{\pgfqpoint{1.175cm}{0.055cm}}{\pgfqpoint{1.21cm}{0.041cm}}{\pgfqpoint{1.246cm}{0.041cm}}
\pgfpathcurveto{\pgfqpoint{1.283cm}{0.041cm}}{\pgfqpoint{1.317cm}{0.055cm}}{\pgfqpoint{1.343cm}{0.081cm}}
\pgfpathcurveto{\pgfqpoint{1.369cm}{0.107cm}}{\pgfqpoint{1.383cm}{0.141cm}}{\pgfqpoint{1.383cm}{0.178cm}}
\pgfusepath{stroke}
\end{pgfscope}
\end{pgfscope}
\end{pgfscope}
\end{pgfscope}
\end{tikzpicture}}} \rangle_{\varepsilon} | \lesssim \lambda^2
     \| \rho^{\sigma} X_{\varepsilon}^{\!\resizebox{!}{.8em}{
\begin{tikzpicture}
\pgfpathmoveto{\pgfqpoint{0cm}{-0.035cm}}
\pgfpathlineto{\pgfqpoint{1.976cm}{-0.035cm}}
\pgfpathlineto{\pgfqpoint{1.976cm}{1.94cm}}
\pgfpathlineto{\pgfqpoint{0cm}{1.94cm}}
\pgfpathclose
\pgfusepath{clip}
\begin{pgfscope}
\begin{pgfscope}
\pgfpathmoveto{\pgfqpoint{0cm}{-0.035cm}}
\pgfpathlineto{\pgfqpoint{1.976cm}{-0.035cm}}
\pgfpathlineto{\pgfqpoint{1.976cm}{1.94cm}}
\pgfpathlineto{\pgfqpoint{0cm}{1.94cm}}
\pgfpathclose
\pgfusepath{clip}
\begin{pgfscope}
\begin{pgfscope}
\pgfsetdash{}{0cm}
\pgfsetlinewidth{0.818mm}
\pgfsetroundcap
\pgfsetroundjoin
\pgfsetmiterlimit{7.0}
\definecolor{eps2pgf_color}{gray}{0}\pgfsetstrokecolor{eps2pgf_color}\pgfsetfillcolor{eps2pgf_color}
\pgfpathmoveto{\pgfqpoint{0.117cm}{1.815cm}}
\pgfpathlineto{\pgfqpoint{0.682cm}{1.065cm}}
\pgfpathlineto{\pgfqpoint{1.246cm}{1.815cm}}
\pgfusepath{stroke}
\end{pgfscope}
\definecolor{eps2pgf_color}{gray}{0}\pgfsetstrokecolor{eps2pgf_color}\pgfsetfillcolor{eps2pgf_color}
\pgfpathmoveto{\pgfqpoint{0.273cm}{1.789cm}}
\pgfpathcurveto{\pgfqpoint{0.273cm}{1.825cm}}{\pgfqpoint{0.259cm}{1.86cm}}{\pgfqpoint{0.233cm}{1.886cm}}
\pgfpathcurveto{\pgfqpoint{0.207cm}{1.912cm}}{\pgfqpoint{0.173cm}{1.926cm}}{\pgfqpoint{0.137cm}{1.926cm}}
\pgfpathcurveto{\pgfqpoint{0.1cm}{1.926cm}}{\pgfqpoint{0.066cm}{1.912cm}}{\pgfqpoint{0.04cm}{1.886cm}}
\pgfpathcurveto{\pgfqpoint{0.014cm}{1.86cm}}{\pgfqpoint{0cm}{1.825cm}}{\pgfqpoint{0cm}{1.789cm}}
\pgfpathcurveto{\pgfqpoint{0cm}{1.753cm}}{\pgfqpoint{0.014cm}{1.718cm}}{\pgfqpoint{0.04cm}{1.692cm}}
\pgfpathcurveto{\pgfqpoint{0.066cm}{1.667cm}}{\pgfqpoint{0.1cm}{1.652cm}}{\pgfqpoint{0.137cm}{1.652cm}}
\pgfpathcurveto{\pgfqpoint{0.173cm}{1.652cm}}{\pgfqpoint{0.207cm}{1.667cm}}{\pgfqpoint{0.233cm}{1.692cm}}
\pgfpathcurveto{\pgfqpoint{0.259cm}{1.718cm}}{\pgfqpoint{0.273cm}{1.753cm}}{\pgfqpoint{0.273cm}{1.789cm}}
\pgfusepath{fill}
\pgfpathmoveto{\pgfqpoint{1.345cm}{1.765cm}}
\pgfpathcurveto{\pgfqpoint{1.345cm}{1.801cm}}{\pgfqpoint{1.331cm}{1.836cm}}{\pgfqpoint{1.305cm}{1.862cm}}
\pgfpathcurveto{\pgfqpoint{1.28cm}{1.887cm}}{\pgfqpoint{1.245cm}{1.902cm}}{\pgfqpoint{1.209cm}{1.902cm}}
\pgfpathcurveto{\pgfqpoint{1.172cm}{1.902cm}}{\pgfqpoint{1.138cm}{1.887cm}}{\pgfqpoint{1.112cm}{1.862cm}}
\pgfpathcurveto{\pgfqpoint{1.087cm}{1.836cm}}{\pgfqpoint{1.072cm}{1.801cm}}{\pgfqpoint{1.072cm}{1.765cm}}
\pgfpathcurveto{\pgfqpoint{1.072cm}{1.728cm}}{\pgfqpoint{1.087cm}{1.694cm}}{\pgfqpoint{1.112cm}{1.668cm}}
\pgfpathcurveto{\pgfqpoint{1.138cm}{1.642cm}}{\pgfqpoint{1.172cm}{1.628cm}}{\pgfqpoint{1.209cm}{1.628cm}}
\pgfpathcurveto{\pgfqpoint{1.245cm}{1.628cm}}{\pgfqpoint{1.28cm}{1.642cm}}{\pgfqpoint{1.305cm}{1.668cm}}
\pgfpathcurveto{\pgfqpoint{1.331cm}{1.694cm}}{\pgfqpoint{1.345cm}{1.728cm}}{\pgfqpoint{1.345cm}{1.765cm}}
\pgfusepath{fill}
\begin{pgfscope}
\pgfsetdash{}{0cm}
\pgfsetlinewidth{0.818mm}
\pgfsetroundcap
\pgfsetroundjoin
\pgfsetmiterlimit{7.0}
\pgfpathmoveto{\pgfqpoint{0.682cm}{1.065cm}}
\pgfpathlineto{\pgfqpoint{1.246cm}{0.315cm}}
\pgfpathlineto{\pgfqpoint{1.811cm}{1.065cm}}
\pgfusepath{stroke}
\end{pgfscope}
\pgfpathmoveto{\pgfqpoint{1.948cm}{1.065cm}}
\pgfpathcurveto{\pgfqpoint{1.948cm}{1.101cm}}{\pgfqpoint{1.933cm}{1.136cm}}{\pgfqpoint{1.907cm}{1.162cm}}
\pgfpathcurveto{\pgfqpoint{1.882cm}{1.187cm}}{\pgfqpoint{1.847cm}{1.202cm}}{\pgfqpoint{1.811cm}{1.202cm}}
\pgfpathcurveto{\pgfqpoint{1.775cm}{1.202cm}}{\pgfqpoint{1.74cm}{1.187cm}}{\pgfqpoint{1.714cm}{1.162cm}}
\pgfpathcurveto{\pgfqpoint{1.689cm}{1.136cm}}{\pgfqpoint{1.674cm}{1.101cm}}{\pgfqpoint{1.674cm}{1.065cm}}
\pgfpathcurveto{\pgfqpoint{1.674cm}{1.029cm}}{\pgfqpoint{1.689cm}{0.994cm}}{\pgfqpoint{1.714cm}{0.968cm}}
\pgfpathcurveto{\pgfqpoint{1.74cm}{0.942cm}}{\pgfqpoint{1.775cm}{0.928cm}}{\pgfqpoint{1.811cm}{0.928cm}}
\pgfpathcurveto{\pgfqpoint{1.847cm}{0.928cm}}{\pgfqpoint{1.882cm}{0.942cm}}{\pgfqpoint{1.907cm}{0.968cm}}
\pgfpathcurveto{\pgfqpoint{1.933cm}{0.994cm}}{\pgfqpoint{1.948cm}{1.029cm}}{\pgfqpoint{1.948cm}{1.065cm}}
\pgfusepath{fill}
\begin{pgfscope}
\pgfsetdash{}{0cm}
\pgfsetlinewidth{0.818mm}
\pgfsetmiterlimit{7.0}
\pgfpathmoveto{\pgfqpoint{1.246cm}{0.315cm}}
\pgfpathlineto{\pgfqpoint{1.244cm}{1.061cm}}
\pgfusepath{stroke}
\end{pgfscope}
\pgfpathmoveto{\pgfqpoint{1.38cm}{1.065cm}}
\pgfpathcurveto{\pgfqpoint{1.38cm}{1.101cm}}{\pgfqpoint{1.366cm}{1.136cm}}{\pgfqpoint{1.34cm}{1.162cm}}
\pgfpathcurveto{\pgfqpoint{1.315cm}{1.187cm}}{\pgfqpoint{1.28cm}{1.202cm}}{\pgfqpoint{1.244cm}{1.202cm}}
\pgfpathcurveto{\pgfqpoint{1.207cm}{1.202cm}}{\pgfqpoint{1.173cm}{1.187cm}}{\pgfqpoint{1.147cm}{1.162cm}}
\pgfpathcurveto{\pgfqpoint{1.121cm}{1.136cm}}{\pgfqpoint{1.107cm}{1.101cm}}{\pgfqpoint{1.107cm}{1.065cm}}
\pgfpathcurveto{\pgfqpoint{1.107cm}{1.029cm}}{\pgfqpoint{1.121cm}{0.994cm}}{\pgfqpoint{1.147cm}{0.968cm}}
\pgfpathcurveto{\pgfqpoint{1.173cm}{0.942cm}}{\pgfqpoint{1.207cm}{0.928cm}}{\pgfqpoint{1.244cm}{0.928cm}}
\pgfpathcurveto{\pgfqpoint{1.28cm}{0.928cm}}{\pgfqpoint{1.315cm}{0.942cm}}{\pgfqpoint{1.34cm}{0.968cm}}
\pgfpathcurveto{\pgfqpoint{1.366cm}{0.994cm}}{\pgfqpoint{1.38cm}{1.029cm}}{\pgfqpoint{1.38cm}{1.065cm}}
\pgfusepath{fill}
\begin{pgfscope}
\pgfsetdash{}{0cm}
\pgfsetlinewidth{0.818mm}
\pgfsetmiterlimit{4.0}
\pgfpathmoveto{\pgfqpoint{1.383cm}{0.178cm}}
\pgfpathcurveto{\pgfqpoint{1.383cm}{0.214cm}}{\pgfqpoint{1.369cm}{0.249cm}}{\pgfqpoint{1.343cm}{0.275cm}}
\pgfpathcurveto{\pgfqpoint{1.317cm}{0.3cm}}{\pgfqpoint{1.283cm}{0.315cm}}{\pgfqpoint{1.246cm}{0.315cm}}
\pgfpathcurveto{\pgfqpoint{1.21cm}{0.315cm}}{\pgfqpoint{1.175cm}{0.3cm}}{\pgfqpoint{1.15cm}{0.275cm}}
\pgfpathcurveto{\pgfqpoint{1.124cm}{0.249cm}}{\pgfqpoint{1.11cm}{0.214cm}}{\pgfqpoint{1.11cm}{0.178cm}}
\pgfpathcurveto{\pgfqpoint{1.11cm}{0.141cm}}{\pgfqpoint{1.124cm}{0.107cm}}{\pgfqpoint{1.15cm}{0.081cm}}
\pgfpathcurveto{\pgfqpoint{1.175cm}{0.055cm}}{\pgfqpoint{1.21cm}{0.041cm}}{\pgfqpoint{1.246cm}{0.041cm}}
\pgfpathcurveto{\pgfqpoint{1.283cm}{0.041cm}}{\pgfqpoint{1.317cm}{0.055cm}}{\pgfqpoint{1.343cm}{0.081cm}}
\pgfpathcurveto{\pgfqpoint{1.369cm}{0.107cm}}{\pgfqpoint{1.383cm}{0.141cm}}{\pgfqpoint{1.383cm}{0.178cm}}
\pgfusepath{stroke}
\end{pgfscope}
\end{pgfscope}
\end{pgfscope}
\end{pgfscope}
\end{tikzpicture}}} \|_{\CC^{- \kappa,
     \varepsilon}}  \| \rho^{4 - \sigma} \phi_{\varepsilon}^2 \|_{B^{\kappa,
     \varepsilon}_{1, 1}} \lesssim \lambda^2 \| \rho^{\sigma}
     X_{\varepsilon}^{\!\resizebox{!}{.8em}{
\begin{tikzpicture}
\pgfpathmoveto{\pgfqpoint{0cm}{-0.035cm}}
\pgfpathlineto{\pgfqpoint{1.976cm}{-0.035cm}}
\pgfpathlineto{\pgfqpoint{1.976cm}{1.94cm}}
\pgfpathlineto{\pgfqpoint{0cm}{1.94cm}}
\pgfpathclose
\pgfusepath{clip}
\begin{pgfscope}
\begin{pgfscope}
\pgfpathmoveto{\pgfqpoint{0cm}{-0.035cm}}
\pgfpathlineto{\pgfqpoint{1.976cm}{-0.035cm}}
\pgfpathlineto{\pgfqpoint{1.976cm}{1.94cm}}
\pgfpathlineto{\pgfqpoint{0cm}{1.94cm}}
\pgfpathclose
\pgfusepath{clip}
\begin{pgfscope}
\begin{pgfscope}
\pgfsetdash{}{0cm}
\pgfsetlinewidth{0.818mm}
\pgfsetroundcap
\pgfsetroundjoin
\pgfsetmiterlimit{7.0}
\definecolor{eps2pgf_color}{gray}{0}\pgfsetstrokecolor{eps2pgf_color}\pgfsetfillcolor{eps2pgf_color}
\pgfpathmoveto{\pgfqpoint{0.117cm}{1.815cm}}
\pgfpathlineto{\pgfqpoint{0.682cm}{1.065cm}}
\pgfpathlineto{\pgfqpoint{1.246cm}{1.815cm}}
\pgfusepath{stroke}
\end{pgfscope}
\definecolor{eps2pgf_color}{gray}{0}\pgfsetstrokecolor{eps2pgf_color}\pgfsetfillcolor{eps2pgf_color}
\pgfpathmoveto{\pgfqpoint{0.273cm}{1.789cm}}
\pgfpathcurveto{\pgfqpoint{0.273cm}{1.825cm}}{\pgfqpoint{0.259cm}{1.86cm}}{\pgfqpoint{0.233cm}{1.886cm}}
\pgfpathcurveto{\pgfqpoint{0.207cm}{1.912cm}}{\pgfqpoint{0.173cm}{1.926cm}}{\pgfqpoint{0.137cm}{1.926cm}}
\pgfpathcurveto{\pgfqpoint{0.1cm}{1.926cm}}{\pgfqpoint{0.066cm}{1.912cm}}{\pgfqpoint{0.04cm}{1.886cm}}
\pgfpathcurveto{\pgfqpoint{0.014cm}{1.86cm}}{\pgfqpoint{0cm}{1.825cm}}{\pgfqpoint{0cm}{1.789cm}}
\pgfpathcurveto{\pgfqpoint{0cm}{1.753cm}}{\pgfqpoint{0.014cm}{1.718cm}}{\pgfqpoint{0.04cm}{1.692cm}}
\pgfpathcurveto{\pgfqpoint{0.066cm}{1.667cm}}{\pgfqpoint{0.1cm}{1.652cm}}{\pgfqpoint{0.137cm}{1.652cm}}
\pgfpathcurveto{\pgfqpoint{0.173cm}{1.652cm}}{\pgfqpoint{0.207cm}{1.667cm}}{\pgfqpoint{0.233cm}{1.692cm}}
\pgfpathcurveto{\pgfqpoint{0.259cm}{1.718cm}}{\pgfqpoint{0.273cm}{1.753cm}}{\pgfqpoint{0.273cm}{1.789cm}}
\pgfusepath{fill}
\pgfpathmoveto{\pgfqpoint{1.345cm}{1.765cm}}
\pgfpathcurveto{\pgfqpoint{1.345cm}{1.801cm}}{\pgfqpoint{1.331cm}{1.836cm}}{\pgfqpoint{1.305cm}{1.862cm}}
\pgfpathcurveto{\pgfqpoint{1.28cm}{1.887cm}}{\pgfqpoint{1.245cm}{1.902cm}}{\pgfqpoint{1.209cm}{1.902cm}}
\pgfpathcurveto{\pgfqpoint{1.172cm}{1.902cm}}{\pgfqpoint{1.138cm}{1.887cm}}{\pgfqpoint{1.112cm}{1.862cm}}
\pgfpathcurveto{\pgfqpoint{1.087cm}{1.836cm}}{\pgfqpoint{1.072cm}{1.801cm}}{\pgfqpoint{1.072cm}{1.765cm}}
\pgfpathcurveto{\pgfqpoint{1.072cm}{1.728cm}}{\pgfqpoint{1.087cm}{1.694cm}}{\pgfqpoint{1.112cm}{1.668cm}}
\pgfpathcurveto{\pgfqpoint{1.138cm}{1.642cm}}{\pgfqpoint{1.172cm}{1.628cm}}{\pgfqpoint{1.209cm}{1.628cm}}
\pgfpathcurveto{\pgfqpoint{1.245cm}{1.628cm}}{\pgfqpoint{1.28cm}{1.642cm}}{\pgfqpoint{1.305cm}{1.668cm}}
\pgfpathcurveto{\pgfqpoint{1.331cm}{1.694cm}}{\pgfqpoint{1.345cm}{1.728cm}}{\pgfqpoint{1.345cm}{1.765cm}}
\pgfusepath{fill}
\begin{pgfscope}
\pgfsetdash{}{0cm}
\pgfsetlinewidth{0.818mm}
\pgfsetroundcap
\pgfsetroundjoin
\pgfsetmiterlimit{7.0}
\pgfpathmoveto{\pgfqpoint{0.682cm}{1.065cm}}
\pgfpathlineto{\pgfqpoint{1.246cm}{0.315cm}}
\pgfpathlineto{\pgfqpoint{1.811cm}{1.065cm}}
\pgfusepath{stroke}
\end{pgfscope}
\pgfpathmoveto{\pgfqpoint{1.948cm}{1.065cm}}
\pgfpathcurveto{\pgfqpoint{1.948cm}{1.101cm}}{\pgfqpoint{1.933cm}{1.136cm}}{\pgfqpoint{1.907cm}{1.162cm}}
\pgfpathcurveto{\pgfqpoint{1.882cm}{1.187cm}}{\pgfqpoint{1.847cm}{1.202cm}}{\pgfqpoint{1.811cm}{1.202cm}}
\pgfpathcurveto{\pgfqpoint{1.775cm}{1.202cm}}{\pgfqpoint{1.74cm}{1.187cm}}{\pgfqpoint{1.714cm}{1.162cm}}
\pgfpathcurveto{\pgfqpoint{1.689cm}{1.136cm}}{\pgfqpoint{1.674cm}{1.101cm}}{\pgfqpoint{1.674cm}{1.065cm}}
\pgfpathcurveto{\pgfqpoint{1.674cm}{1.029cm}}{\pgfqpoint{1.689cm}{0.994cm}}{\pgfqpoint{1.714cm}{0.968cm}}
\pgfpathcurveto{\pgfqpoint{1.74cm}{0.942cm}}{\pgfqpoint{1.775cm}{0.928cm}}{\pgfqpoint{1.811cm}{0.928cm}}
\pgfpathcurveto{\pgfqpoint{1.847cm}{0.928cm}}{\pgfqpoint{1.882cm}{0.942cm}}{\pgfqpoint{1.907cm}{0.968cm}}
\pgfpathcurveto{\pgfqpoint{1.933cm}{0.994cm}}{\pgfqpoint{1.948cm}{1.029cm}}{\pgfqpoint{1.948cm}{1.065cm}}
\pgfusepath{fill}
\begin{pgfscope}
\pgfsetdash{}{0cm}
\pgfsetlinewidth{0.818mm}
\pgfsetmiterlimit{7.0}
\pgfpathmoveto{\pgfqpoint{1.246cm}{0.315cm}}
\pgfpathlineto{\pgfqpoint{1.244cm}{1.061cm}}
\pgfusepath{stroke}
\end{pgfscope}
\pgfpathmoveto{\pgfqpoint{1.38cm}{1.065cm}}
\pgfpathcurveto{\pgfqpoint{1.38cm}{1.101cm}}{\pgfqpoint{1.366cm}{1.136cm}}{\pgfqpoint{1.34cm}{1.162cm}}
\pgfpathcurveto{\pgfqpoint{1.315cm}{1.187cm}}{\pgfqpoint{1.28cm}{1.202cm}}{\pgfqpoint{1.244cm}{1.202cm}}
\pgfpathcurveto{\pgfqpoint{1.207cm}{1.202cm}}{\pgfqpoint{1.173cm}{1.187cm}}{\pgfqpoint{1.147cm}{1.162cm}}
\pgfpathcurveto{\pgfqpoint{1.121cm}{1.136cm}}{\pgfqpoint{1.107cm}{1.101cm}}{\pgfqpoint{1.107cm}{1.065cm}}
\pgfpathcurveto{\pgfqpoint{1.107cm}{1.029cm}}{\pgfqpoint{1.121cm}{0.994cm}}{\pgfqpoint{1.147cm}{0.968cm}}
\pgfpathcurveto{\pgfqpoint{1.173cm}{0.942cm}}{\pgfqpoint{1.207cm}{0.928cm}}{\pgfqpoint{1.244cm}{0.928cm}}
\pgfpathcurveto{\pgfqpoint{1.28cm}{0.928cm}}{\pgfqpoint{1.315cm}{0.942cm}}{\pgfqpoint{1.34cm}{0.968cm}}
\pgfpathcurveto{\pgfqpoint{1.366cm}{0.994cm}}{\pgfqpoint{1.38cm}{1.029cm}}{\pgfqpoint{1.38cm}{1.065cm}}
\pgfusepath{fill}
\begin{pgfscope}
\pgfsetdash{}{0cm}
\pgfsetlinewidth{0.818mm}
\pgfsetmiterlimit{4.0}
\pgfpathmoveto{\pgfqpoint{1.383cm}{0.178cm}}
\pgfpathcurveto{\pgfqpoint{1.383cm}{0.214cm}}{\pgfqpoint{1.369cm}{0.249cm}}{\pgfqpoint{1.343cm}{0.275cm}}
\pgfpathcurveto{\pgfqpoint{1.317cm}{0.3cm}}{\pgfqpoint{1.283cm}{0.315cm}}{\pgfqpoint{1.246cm}{0.315cm}}
\pgfpathcurveto{\pgfqpoint{1.21cm}{0.315cm}}{\pgfqpoint{1.175cm}{0.3cm}}{\pgfqpoint{1.15cm}{0.275cm}}
\pgfpathcurveto{\pgfqpoint{1.124cm}{0.249cm}}{\pgfqpoint{1.11cm}{0.214cm}}{\pgfqpoint{1.11cm}{0.178cm}}
\pgfpathcurveto{\pgfqpoint{1.11cm}{0.141cm}}{\pgfqpoint{1.124cm}{0.107cm}}{\pgfqpoint{1.15cm}{0.081cm}}
\pgfpathcurveto{\pgfqpoint{1.175cm}{0.055cm}}{\pgfqpoint{1.21cm}{0.041cm}}{\pgfqpoint{1.246cm}{0.041cm}}
\pgfpathcurveto{\pgfqpoint{1.283cm}{0.041cm}}{\pgfqpoint{1.317cm}{0.055cm}}{\pgfqpoint{1.343cm}{0.081cm}}
\pgfpathcurveto{\pgfqpoint{1.369cm}{0.107cm}}{\pgfqpoint{1.383cm}{0.141cm}}{\pgfqpoint{1.383cm}{0.178cm}}
\pgfusepath{stroke}
\end{pgfscope}
\end{pgfscope}
\end{pgfscope}
\end{pgfscope}
\end{tikzpicture}}} \|_{\CC^{- \kappa, \varepsilon}}  \| \rho^{1
     + \iota} \phi_{\varepsilon} \|_{L^{2, \varepsilon}}  \| \rho^{3 - \iota -
     \sigma} \phi_{\varepsilon} \|_{H^{2 \kappa, \varepsilon}} \]
  \[ \lesssim \lambda^2 \|\mathbb{X}_{\varepsilon} \|^4  \| \rho
     \phi_{\varepsilon} \|^{1 + \theta}_{L^{4, \varepsilon}}  \| \rho^2
     \phi_{\varepsilon} \|^{1 - \theta}_{H^{1 - 2 \kappa, \varepsilon}}
     \leqslant \lambda^{(7 - \theta) / (1 + \theta)} C_{\rho}
     \|\mathbb{X}_{\varepsilon} \|^{8 + \vartheta} + \delta (\lambda \| \rho
     \phi_{\varepsilon} \|_{L^{4, \varepsilon}}^4 +\| \rho^2
     \phi_{\varepsilon} \|_{H^{1 - 2 \kappa, \varepsilon}}^2) . \]
  Recall that since $\sigma$ is chosen small, we have the interpolation
  inequality (see Lemma~\ref{lem:int})
  \[ \| \phi_{\varepsilon} \|_{H^{1 / 2 + \kappa, \varepsilon} (\rho^{2 -
     \sigma / 2})} \leqslant \| \phi_{\varepsilon} \|_{L^{2, \varepsilon}
     (\rho^{1 + \iota})}^{\theta} \| \phi_{\varepsilon} \|_{H^{1 - 2 \kappa,
     \varepsilon} (\rho^2)}^{1 - \theta} \]
  where $\theta = \frac{1 / 2 - 3 \kappa}{1 - 2 \kappa}$. Similar
  interpolation inequalities will also be employed below. Then, in view of
  Lemma~\ref{lem:dual1} and Young's inequality, we have
  \[ \lambda |D_{\rho^4, \varepsilon} (\phi_{\varepsilon}, - 3 \llbracket
     X_{\varepsilon}^2 \rrbracket, \phi_{\varepsilon}) | \lesssim \lambda \|
     \rho^{\sigma} \llbracket X_{\varepsilon}^2 \rrbracket \|_{\CC^{- 1 -
     \kappa, \varepsilon}}  \| \rho^{2 - \sigma / 2} \phi_{\varepsilon}
     \|_{H^{1 / 2 + \kappa, \varepsilon}}^2 \]
  \[ \lesssim \lambda \| \rho^{\sigma} \llbracket X_{\varepsilon}^2 \rrbracket
     \|_{\CC^{- 1 - \kappa, \varepsilon}} \| \rho^{1 + \iota}
     \phi_{\varepsilon} \|_{L^{2, \varepsilon}}^{2 \theta}  \| \rho^2
     \phi_{\varepsilon} \|_{H^{1 - 2 \kappa, \varepsilon}}^{2 (1 - \theta)} \]
  \[ \lesssim \lambda \|\mathbb{X}_{\varepsilon} \|^2  \| \rho
     \phi_{\varepsilon} \|_{L^{4, \varepsilon}}^{2 \theta}  \| \rho^2
     \phi_{\varepsilon} \|_{H^{1 - 2 \kappa, \varepsilon}}^{2 (1 - \theta)}
     \leqslant \lambda^{2 / \theta - 1} C_{\delta} \|\mathbb{X}_{\varepsilon}
     \|^{8 + \vartheta} + \delta (\lambda \| \rho \phi_{\varepsilon} \|_{L^{4,
     \varepsilon}}^4 +\| \rho^2 \phi_{\varepsilon} \|_{H^{1 - 2 \kappa,
     \varepsilon}}^2) . \]
  Similarly,
  \[ \lambda^2 \left| D_{\rho^4, \varepsilon} \left( \phi_{\varepsilon}, 3
     \llbracket X_{\varepsilon}^2 \rrbracket, \Q_{\varepsilon}^{- 1} [3
     \llbracket X_{\varepsilon}^2 \rrbracket \succ \phi_{\varepsilon}] \right)
     \right| \]
  \[ \lesssim \lambda^2  \| \rho^{\sigma} \llbracket X_{\varepsilon}^2
     \rrbracket \|_{\CC^{- 1 - \kappa, \varepsilon}} \| \rho^{3 - \iota - 2
     \sigma} \phi_{\varepsilon} \|_{H^{4 \kappa, \varepsilon}} \left| \rho^{1
     + \iota + \sigma} \Q_{\varepsilon}^{- 1} [3 \llbracket X_{\varepsilon}^2
     \rrbracket \succ \phi_{\varepsilon}] \right|_{H^{1 - 2 \kappa,
     \varepsilon}}, \]
  where we further estimate by Schauder and paraproduct estimates
  \[ \left| \rho^{1 + \iota + \sigma} \Q_{\varepsilon}^{- 1} [3 \llbracket
     X_{\varepsilon}^2 \rrbracket \succ \phi_{\varepsilon}] \right|_{H^{1 - 2
     \kappa, \varepsilon}} \lesssim \| \rho^{1 + \iota + \sigma} \llbracket
     X_{\varepsilon}^2 \rrbracket \succ \phi_{\varepsilon} \|_{H^{- 1 - 2
     \kappa, \varepsilon}} \]
  \[ \lesssim \| \rho^{\sigma} \llbracket X_{\varepsilon}^2 \rrbracket
     \|_{\CC^{- 1 - \kappa, \varepsilon}}  \| \rho^{1 + \iota}
     \phi_{\varepsilon} \|_{L^{2, \varepsilon}} \]
  and hence we deduce by interpolation with $\theta = \frac{1 - 6 \kappa}{1 -
  2 \kappa}$ and embedding that
  \[ \lambda^2 \left| D_{\rho^4, \varepsilon} \left( \phi_{\varepsilon}, 3
     \llbracket X_{\varepsilon}^2 \rrbracket, \Q_{\varepsilon}^{- 1} [3
     \llbracket X_{\varepsilon}^2 \rrbracket \succ \phi_{\varepsilon}] \right)
     \right| \lesssim \lambda^2 \|\mathbb{X}_{\varepsilon} \|^4  \| \rho^{1 +
     \iota} \phi_{\varepsilon} \|_{L^{2, \varepsilon}}  \| \rho^2
     \phi_{\varepsilon} \|_{H^{4 \kappa, \varepsilon}} \]
  \[ \lesssim \lambda^2 \|\mathbb{X}_{\varepsilon} \|^4  \| \rho
     \phi_{\varepsilon} \|^{1 + \theta}_{L^{2, \varepsilon}}  \| \rho^2
     \phi_{\varepsilon} \|^{1 - \theta}_{H^{1 - 2 \kappa, \varepsilon}} \]
  \[ \leqslant \lambda^{(7 - \theta) / (1 + \theta)} C_{\delta}
     \|\mathbb{X}_{\varepsilon} \|^{8 + \vartheta} + \delta (\lambda \| \rho
     \phi_{\varepsilon} \|_{L^{4, \varepsilon}}^4 +\| \rho^2
     \phi_{\varepsilon} \|_{H^{1 - 2 \kappa, \varepsilon}}^2) . \]
  Due to Lemma~\ref{lem:comm1} and interpolation with $\theta =
  \frac{1 - 5 \kappa}{1 - 2 \kappa}$, we obtain 
  \[ \lambda^2 | \langle \rho^4 \phi_{\varepsilon}, \tilde{C}
     (\phi_{\varepsilon}, 3 \llbracket X_{\varepsilon}^2 \rrbracket, 3
     \llbracket X_{\varepsilon}^2 \rrbracket) \rangle_{\varepsilon} | \lesssim
     \lambda^2 \| \rho^{\sigma} \llbracket X_{\varepsilon}^2 \rrbracket
     \|_{\CC^{- 1 - \kappa, \varepsilon}}^2  \| \rho^{2 - \sigma}
     \phi_{\varepsilon} \|^2_{H^{3 \kappa, \varepsilon}} \]
  \[ \leqslant \lambda^2 C_{\delta} \|\mathbb{X}_{\varepsilon} \|^4  \|
     \rho^{1 + \iota} \phi_{\varepsilon} \|^{2 \theta}_{L^{2, \varepsilon}} 
     \| \rho^2 \phi_{\varepsilon} \|^{2 (1 - \theta)}_{H^{1 - 2 \kappa,
     \varepsilon}} \]
  \[ \leqslant \lambda^{4 / \theta - 1} C_{\delta} \|\mathbb{X}_{\varepsilon}
     \|^{8 + \vartheta} + \delta (\lambda \| \rho \phi_{\varepsilon} \|_{L^{4,
     \varepsilon}}^4 +\| \rho^2 \phi_{\varepsilon} \|_{H^{1 - 2 \kappa,
     \varepsilon}}^2) . \]
  Then we use the paraproduct estimates, the embedding $\CC^{1 / 2 - \kappa,
  \varepsilon} (\rho^{\sigma}) \subset H^{1 / 2 - 2 \kappa, \varepsilon}
  (\rho^{2 - \sigma / 2})$ (which holds due to the integrability of $\rho^{4
  \iota}$ for some $\iota \in (0, 1)$ and the fact that $\sigma$ can be chosen
  small), together with Lemma~\ref{lem:Y1} and interpolation to deduce for
  $\theta = \frac{1 / 2 - 5 \kappa}{1 - 2 \kappa}$ that
  \[ \lambda | \langle \rho^4 \phi_{\varepsilon}, - 3 \llbracket
     X_{\varepsilon}^2 \rrbracket \prec (Y_{\varepsilon} + \phi_{\varepsilon})
     \rangle_{\varepsilon} | \]
  \[ \lesssim \lambda \| \rho^{\sigma} \llbracket X_{\varepsilon}^2 \rrbracket
     \|_{\CC^{- 1 - \kappa, \varepsilon}} \| \rho^{2 - \sigma / 2}
     (Y_{\varepsilon} + \phi_{\varepsilon})\|_{H^{1 / 2 - 2 \kappa,
     \varepsilon}}  \| \rho^{2 - \sigma / 2} \phi_{\varepsilon} \|_{H^{1 / 2 +
     3 \kappa, \varepsilon}} \]
  \[ \lesssim \lambda \| \rho^{\sigma} \llbracket X_{\varepsilon}^2 \rrbracket
     \|_{\CC^{- 1 - \kappa, \varepsilon}} \| \rho^{2 - \sigma / 2}
     Y_{\varepsilon} \|_{H^{1 / 2 - 2 \kappa, \varepsilon}}  \| \rho^{2 -
     \sigma / 2} \phi_{\varepsilon} \|_{H^{1 / 2 + 3 \kappa, \varepsilon}} \]
  \[ + \lambda \| \rho^{\sigma} \llbracket X_{\varepsilon}^2 \rrbracket
     \|_{\CC^{- 1 - \kappa, \varepsilon}}  \| \rho^{2 - \sigma / 2}
     \phi_{\varepsilon} \|_{H^{1 / 2 + 3 \kappa, \varepsilon}}^2 \]
  \[ \lesssim \lambda  (\lambda \|\mathbb{X}_{\varepsilon} \|^5 \|
     \rho^{1 + \iota} \phi_{\varepsilon} \|_{L^{2, \varepsilon}}^{\theta} \|
     \rho^2 \phi_{\varepsilon} \|_{H^{1 - 2 \kappa, \varepsilon}}^{1 - \theta}
     +\|\mathbb{X}_{\varepsilon} \|^2 \| \rho^{1 + \iota} \phi_{\varepsilon}
     \|_{L^{2, \varepsilon}}^{2 \theta} \| \rho^2 \phi_{\varepsilon} \|_{H^{1
     - 2 \kappa, \varepsilon}}^{2 (1 - \theta)}) \]
  \[ \leqslant (\lambda^{(8 - \theta) / (2 + \theta)} + \lambda^{2 / \theta -
     1}) C_{\delta} \|\mathbb{X}_{\varepsilon} \|^{8 + \vartheta} + \delta
     (\lambda \| \rho \phi_{\varepsilon} \|_{L^{4, \varepsilon}}^4 +\| \rho^2
     \phi_{\varepsilon} \|_{H^{1 - 2 \kappa, \varepsilon}}^2) . \]
  Next, we have
  \[ \lambda | \langle \rho^4 \phi_{\varepsilon}, - 3 X_{\varepsilon}
     (Y_{\varepsilon} + \phi_{\varepsilon})^2 \rangle_{\varepsilon} | \lesssim
     \lambda \| \rho^{\sigma} X_{\varepsilon} \|_{\CC^{- 1 / 2 - \kappa,
     \varepsilon}}  \| \rho^{4 - \sigma} \phi_{\varepsilon}^3 \|_{B_{1, 1}^{1
     / 2 + \kappa, \varepsilon}} \]
  \[ + \lambda \| \rho^{\sigma} X_{\varepsilon} Y_{\varepsilon} \|_{\CC^{- 1 /
     2 - \kappa, \varepsilon}}  \| \rho^{4 - \sigma} \phi_{\varepsilon}^2
     \|_{B_{1, 1}^{1 / 2 + \kappa, \varepsilon}} + \lambda \| \rho^{\sigma}
     X_{\varepsilon} Y_{\varepsilon}^2 \|_{\CC^{- 1 / 2 - \kappa,
     \varepsilon}}  \| \rho^{4 - \sigma} \phi_{\varepsilon} \|_{B_{1, 1}^{1 /
     2 + \kappa, \varepsilon}} . \]
  Here we employ Lemma~\ref{lem:mult} and interpolation to obtain for $\theta
  = \frac{1 / 2 - 4 \kappa}{1 - 2 \kappa}$
  \[ \lambda \| \rho^{\sigma} X_{\varepsilon} \|_{\CC^{- 1 / 2 - \kappa,
     \varepsilon}}  \| \rho^{4 - \sigma} \phi_{\varepsilon}^3 \|_{B_{1, 1}^{1
     / 2 + \kappa, \varepsilon}} \lesssim \lambda \| \rho^{\sigma}
     X_{\varepsilon} \|_{\CC^{- 1 / 2 - \kappa, \varepsilon}}  \| \rho
     \phi_{\varepsilon} \|^2_{L^{4, \varepsilon}}  \| \rho^{2 - \sigma}
     \phi_{\varepsilon} \|_{H^{1 / 2 + 2 \kappa, \varepsilon}} \]
  \[ \lesssim \lambda \|\mathbb{X}_{\varepsilon} \|  \| \rho
     \phi_{\varepsilon} \|_{L^{4, \varepsilon}}^{2 + \theta}  \| \rho^2
     \phi_{\varepsilon} \|_{H^{1 - 2 \kappa, \varepsilon}}^{1 - \theta}
     \leqslant \lambda^{(2 - \theta) / \theta} C_{\delta}
     \|\mathbb{X}_{\varepsilon} \|^{8 + \vartheta} + \delta (\lambda \| \rho
     \phi_{\varepsilon} \|_{L^{4, \varepsilon}}^4 +\| \rho^2
     \phi_{\varepsilon} \|_{H^{1 - 2 \kappa, \varepsilon}}^2) \]
  and similarly for the other two terms, where we also use Lemma~\ref{lem:Z}
  and the embedding $H^{1 - 2 \kappa, \varepsilon} (\rho^2) \subset H^{1 / 2 +
  2 \kappa, \varepsilon} (\rho^{3 - \iota - \sigma})$ and $H^{1 / 2 + 2
  \kappa, \varepsilon} (\rho^2) = B_{2, 2}^{1 / 2 + 2 \kappa, \varepsilon}
  (\rho^2) \subset B_{1, 1}^{1 / 2 + \kappa, \varepsilon} (\rho^{4 - \sigma})$
  together with interpolation with $\theta = \frac{1 / 2 - 4 \kappa}{1 - 2
  \kappa}$
  \[ \lambda \| \rho^{\sigma} X_{\varepsilon} Y_{\varepsilon} \|_{\CC^{- 1 / 2
     - \kappa, \varepsilon}}  \| \rho^{4 - \sigma} \phi_{\varepsilon}^2
     \|_{B_{1, 1}^{1 / 2 + \kappa, \varepsilon}} + \lambda \| \rho^{\sigma}
     X_{\varepsilon} Y_{\varepsilon}^2 \|_{\CC^{- 1 / 2 - \kappa,
     \varepsilon}}  \| \rho^{4 - \sigma} \phi_{\varepsilon} \|_{B_{1, 1}^{1 /
     2 + \kappa, \varepsilon}} \]
  \[ \lesssim (\lambda^2 + \lambda^3) \|\mathbb{X}_{\varepsilon} \|^6  \|
     \rho^{1 + \iota} \phi_{\varepsilon} \|_{L^{2, \varepsilon}}  \| \rho^{3 -
     \iota - \sigma} \phi_{\varepsilon} \|_{H^{1 / 2 + 2 \kappa, \varepsilon}}
     + (\lambda^3 + \lambda^4) \|\mathbb{X}_{\varepsilon} \|^9  \| \rho^2
     \phi_{\varepsilon} \|_{H^{1 / 2 + 2 \kappa, \varepsilon}} \]
  \[ \lesssim (\lambda^2 + \lambda^3) \|\mathbb{X}_{\varepsilon} \|^6  \| \rho
     \phi_{\varepsilon} \|^{1 + \theta}_{L^{4, \varepsilon}}  \| \rho^2
     \phi_{\varepsilon} \|^{1 - \theta}_{H^{1 - 2 \kappa, \varepsilon}} +
     (\lambda^3 + \lambda^4) \|\mathbb{X}_{\varepsilon} \|^9  \| \rho
     \phi_{\varepsilon} \|^{\theta}_{L^{4, \varepsilon}}  \| \rho^2
     \phi_{\varepsilon} \|^{1 - \theta}_{H^{1 - 2 \kappa, \varepsilon}} \]
  \begin{equation}
  \label{eq:XY}
   \leqslant (\lambda^{(11 - \theta) / (2 + \theta)} + \lambda^{(12 -
     \theta) / (2 + \theta)}) C_{\delta} \|\mathbb{X}_{\varepsilon} \|^{16 +
     \vartheta} + \delta (\lambda \| \rho \phi_{\varepsilon} \|_{L^{4,
     \varepsilon}}^4 +\| \rho^2 \phi_{\varepsilon} \|_{H^{1 - 2 \kappa,
     \varepsilon}}^2) .
     \end{equation}
  Next, we obtain
  \begin{equation}\label{eq:Y3}
   \lambda | \langle \rho^4 \phi_{\varepsilon}, - Y_{\varepsilon}^3
     \rangle_{\varepsilon} | \lesssim \lambda \| \rho^{\sigma} Y_{\varepsilon}
     \|_{L^{\infty, \varepsilon}}^3  \| \rho^{4 - 3 \sigma} \phi_{\varepsilon}
     \|_{L^{1, \varepsilon}} \lesssim \lambda^4 \|\mathbb{X}_{\varepsilon}
     \|^9  \| \rho \phi_{\varepsilon} \|_{L^{4, \varepsilon}} \leqslant
     \lambda^5 C_{\delta} \|\mathbb{X}_{\varepsilon} \|^{12} + \delta \lambda
     \| \rho \phi_{\varepsilon} \|^4_{L^{4, \varepsilon}}, 
     \end{equation}
  and similarly
  \[ \lambda | \langle \rho^4 \phi_{\varepsilon}, - 3 Y_{\varepsilon}^2
     \phi_{\varepsilon} \rangle_{\varepsilon} | \lesssim \lambda \|
     \rho^{\sigma} Y_{\varepsilon} \|_{L^{\infty, \varepsilon}}^2  \| \rho^{4
     - \sigma} \phi_{\varepsilon}^2 \|_{L^{1, \varepsilon}} \]
    \begin{equation}\label{eq:Y2}
     \lesssim \lambda^3 \|\mathbb{X}_{\varepsilon} \|^6  \| \rho
     \phi_{\varepsilon} \|^2_{L^{4, \varepsilon}} \leqslant \lambda^5
     C_{\delta} \|\mathbb{X}_{\varepsilon} \|^{12} + \delta \lambda \| \rho
     \phi_{\varepsilon} \|_{L^{4, \varepsilon}}^4, 
     \end{equation}
  \[ \lambda | \langle \rho^4 \phi_{\varepsilon}, - 3 Y_{\varepsilon}
     \phi_{\varepsilon}^2 \rangle_{\varepsilon} | \lesssim \lambda \|
     \rho^{\sigma} Y_{\varepsilon} \|_{L^{\infty, \varepsilon}}  \| \rho^{4 -
     \sigma} \phi_{\varepsilon}^3 \|_{L^{1, \varepsilon}} \lesssim \lambda \|
     \rho^{\sigma} Y_{\varepsilon} \|_{L^{\infty, \varepsilon}}  \| \rho
     \phi_{\varepsilon} \|^3_{L^{4, \varepsilon}} \]
  \begin{equation}\label{eq:Y11}
   \lesssim \lambda^2 \|\mathbb{X}_{\varepsilon} \|^3  \| \rho
     \phi_{\varepsilon} \|_{L^{4, \varepsilon}}^3 \leqslant \lambda^5
     C_{\delta} \|\mathbb{X}_{\varepsilon} \|^{12} + \delta \lambda \| \rho
     \phi_{\varepsilon} \|_{L^{4, \varepsilon}}^4 .
     \end{equation}
  Then, by {\eqref{eq:U11}}
  \[ \lambda \left| \langle \rho^4 \phi_{\varepsilon}, - 3
     (\UU^{\varepsilon}_{\leqslant} \llbracket X^2 \rrbracket) \succ
     Y_{\varepsilon} \rangle_{\varepsilon} \right| \lesssim \lambda \|
     \rho^{\sigma} \UU^{\varepsilon}_{\leqslant} \llbracket X_{\varepsilon}^2
     \rrbracket \|_{\CC^{- 1 + 3 \kappa, \varepsilon}} \| \rho^{\sigma}
     Y_{\varepsilon} \|_{L^{\infty, \varepsilon}}  \| \rho^{4 - 3 \sigma}
     \phi_{\varepsilon} \|_{B^{1 - 3 \kappa, \varepsilon}_{1, 1}} \]
  \[ \lesssim \lambda (1 + \lambda \| \rho^{\sigma} \llbracket
     X_{\varepsilon}^2 \rrbracket \|_{\CC^{- 1 - \kappa, \varepsilon}})^{8
     \kappa}  \| \rho^{\sigma} \llbracket X_{\varepsilon}^2 \rrbracket
     \|_{\CC^{- 1 - \kappa, \varepsilon}} \| \rho^{\sigma} Y_{\varepsilon}
     \|_{L^{\infty, \varepsilon}}  \| \rho^2 \phi_{\varepsilon} \|_{H^{1 - 2
     \kappa, \varepsilon}} \]
  \begin{equation}\label{eq:U0}
   \lesssim (\lambda^{2} + \lambda^{2 + 8 \kappa}) \|\mathbb{X}_{\varepsilon}
     \|^{5 + 16 \kappa}  \| \rho^2 \phi_{\varepsilon} \|_{H^{1 - 2 \kappa,
     \varepsilon}} \leqslant (\lambda^4 + \lambda^5) C_{\delta}
     \|\mathbb{X}_{\varepsilon} \|^{10 + \vartheta} + \delta \| \rho^2
     \phi_{\varepsilon} \|_{H^{1 - 2 \kappa, \varepsilon}}^2,
     \end{equation}
  and finally for $\theta = \frac{1 / 2 - 4 \kappa}{1 - 2 \kappa}$
  \[ \lambda^2 | \langle \rho^4 \phi_{\varepsilon}, Z_{\varepsilon}
     \rangle_{\varepsilon} | \lesssim \lambda^2  \| \rho^{\sigma}
     Z_{\varepsilon} \|_{\CC^{- 1 / 2 - \kappa, \varepsilon}}  \| \rho^{4 -
     \sigma} \phi_{\varepsilon} \|_{B_{1, 1}^{1 / 2 + \kappa, \varepsilon}} \]
  \[ \lesssim (\lambda^2 + \lambda^3 | \log t| + \lambda^4)
     \|\mathbb{X}_{\varepsilon} \|^{7+\vartheta}  \| \rho \phi_{\varepsilon} \|_{L^{4,
     \varepsilon}}^{\theta}  \| \rho^2 \phi \|_{H^{1 - 2 \kappa}}^{1 - \theta}
  \]
 \[
   \leqslant (\lambda^{(8 - \theta) / (2 + \theta)} + \lambda^{(12 - \theta)
     / (2 + \theta)} | \log t|^{4 / (2 + \theta)} + \lambda^{(16 - \theta) /
     (2 + \theta)}) C_{\delta} \|\mathbb{X}_{\varepsilon} \|^{12} \]
   \begin{equation}\label{eq:Z0}
    + \delta (\lambda \| \rho \phi_{\varepsilon} \|_{L^{4, \varepsilon}}^4
     +\| \rho^2 \phi_{\varepsilon} \|_{H^{1 - 2 \kappa, \varepsilon}}^2) .
     \end{equation}
  The proof is complete.
\end{proof}

Now we have all in hand to establish our main energy estimate.

\begin{theorem}
  \label{th:energy-estimate}
  Let $\rho$ be a weight such that $\rho^{\iota} \in
  L^{4, 0}$ for some $\iota \in (0, 1)$. There exists a constant $\alpha=\alpha(m^{2}) \in (0,1)$  such that for    $\theta=\frac{1/2-4\kappa}{1-2\kappa}$
  \begin{equation}
    \frac{1}{2} \partial_t \| \rho^2 \phi_{\varepsilon} \|_{L^{2,
    \varepsilon}}^2 + \alpha [\lambda \| \rho \phi_{\varepsilon} \|_{L^{4,
    \varepsilon}}^4 +  m^2 \| \rho^2 \psi_{\varepsilon} \|_{L^{2,
    \varepsilon}}^2 +  \| \rho^2 \nabla_{\varepsilon} \psi_{\varepsilon}
    \|_{L^{2, \varepsilon}}^2] + \| \rho^2 \phi_{\varepsilon} \|_{H^{1 - 2
    \kappa, \varepsilon}}^2 \label{eq:d18}
  \end{equation}
  \[ \leqslant (\lambda^3 + \lambda^{(12 - \theta)
     / (2 + \theta)} | \log t|^{4 / (2 + \theta)} + \lambda^{7}) Q_{\rho} (\mathbb{X}_{\varepsilon}) . \]
\end{theorem}

\begin{proof}
  As a consequence of {\eqref{eq:psi1}}, we have according to Lemma~\ref{lem:grad}, Lemma~\ref{lem:emb}, Lemma~\ref{lem:equiv2}
  \[ \| \rho^2 \phi_{\varepsilon} \|_{H^{1 - 2 \kappa, \varepsilon}}^2
     \lesssim \left\| \rho^2 \Q_{\varepsilon}^{- 1} [3\lambda\llbracket
     X_{\varepsilon}^2 \rrbracket \succ \phi_{\varepsilon}] \right\|_{H^{1 - 2
     \kappa, \varepsilon}}^2 + \| \rho^2 \psi_{\varepsilon} \|_{H^{1 - 2
     \kappa, \varepsilon}}^2 \]
  \[ \lesssim \lambda^2 \| \rho^{\sigma} \llbracket X_{\varepsilon}^2 \rrbracket
     \|^2_{\CC^{- 1 - \kappa, \varepsilon}} \| \rho^{2 - \sigma}
     \phi_{\varepsilon} \|^2_{L^{2, \varepsilon}} + \| \rho^2
     \psi_{\varepsilon} \|_{H^{1 - \kappa, \varepsilon}}^2 \]
  \begin{equation}
    \lesssim \lambda^3 Q_{\rho} (\mathbb{X}_{\varepsilon}) + \lambda \| \rho \phi_{\varepsilon}
    \|_{L^{4, \varepsilon}}^4 + \| \rho^2 \psi_{\varepsilon} \|_{L^{2,
    \varepsilon}}^2 + \| \rho^2 \nabla_{\varepsilon} \psi_{\varepsilon}
    \|_{L^{2, \varepsilon}}^2 . \label{eq:17}
  \end{equation}
  Therefore, according to Lemma~\ref{lemma:bounds-rhs1} we obtain that
  \[ \frac{1}{2} \partial_t \| \rho^2 \phi_{\varepsilon} \|_{L^{2,
     \varepsilon}}^2 +\lambda \| \rho \phi_{\varepsilon} \|_{L^{4, \varepsilon}}^4 +
     m^{2} \| \rho^2 \psi_{\varepsilon} \|_{L^{2, \varepsilon}}^2 + \| \rho^2
     \nabla_{\varepsilon} \psi_{\varepsilon} \|_{L^{2, \varepsilon}}^2  \]
  \[ \leqslant(\lambda^3 + \lambda^{(12 - \theta)
     / (2 + \theta)} | \log t|^{4 / (2 + \theta)} + \lambda^{7})Q_{\rho} (\mathbb{X}_{\varepsilon}) + \delta C
     (\lambda\| \rho \phi_{\varepsilon} \|_{L^{4, \varepsilon}}^4 + \| \rho^2
     \psi_{\varepsilon} \|_{L^{2, \varepsilon}}^2 + \| \rho^2
     \nabla_{\varepsilon} \psi_{\varepsilon} \|_{L^{2, \varepsilon}}^2). \]
Choosing $\delta > 0$ sufficiently small (depending  on $m^2$ and the implicit constant $C$ from Lemma~\ref{lem:grad}) allows to absorb the norms of
  $\phi_{\varepsilon}, \psi_{\varepsilon}$ from the right hand side into the
  left hand side and the claim follows.
\end{proof}

\begin{remark}\label{rem:neg-mass}
We point out that the requirement of a strictly positive mass $m^{2}>0$ is to some extent superfluous for our approach. To be more precise, if $m^{2}\leqslant 0$ then we may rewrite the mollified stochastic quantization equation as
$$
(\partial_{t}-\Delta_{\varepsilon} +1)\varphi_{\varepsilon} +\lambda\varphi_{\varepsilon}^{3}=\xi_{\varepsilon}+(1-m^{2})\varphi_{\varepsilon}
$$
and the same decomposition as above introduces an additional term on the right hand side of \eqref{eq:en12}. This can be controlled by
$$
|(1-m^{2})\langle \rho^{4}\phi_{\varepsilon},X_{\varepsilon}+Y_{\varepsilon}+\phi_{\varepsilon}\rangle|\lesssim C_{\delta,\lambda^{-1}}Q_{\rho}(\mathbb{X}_{\varepsilon})+\delta(\lambda\|\rho\phi_{\varepsilon}\|_{L^{4,\varepsilon}}^{4}+\|\rho^{2}\phi_{\varepsilon}\|_{H^{1-2\kappa,\varepsilon}}^{2}),
$$
where we write $C_{\delta,\lambda^{-1}}$ to stress that the constant is not uniform over small $\lambda$.
As a consequence, we obtain an analogue of Theorem \ref{th:energy-estimate} but the uniformity  for small $\lambda$ is not valid anymore.
\end{remark}

\begin{corollary}
  \label{cor:Lp}Let $\rho$ be a weight such that $\rho^{\iota} \in L^{4, 0}$
  for some $\iota \in (0, 1)$. Then
  for all $p \in [1, \infty)$ and    $\theta=\frac{1/2-4\kappa}{1-2\kappa}$
  \begin{equation}\label{eq:lp} \frac{1}{2 p} \partial_t \| \rho^2 \phi_{\varepsilon} \|_{L^{2,
     \varepsilon}}^{2 p} + \lambda \| \rho^2 \phi_{\varepsilon} \|_{L^{2,
     \varepsilon}}^{2 p + 2} \leqslant  \lambda [(\lambda^{2} + \lambda^{(10 - 2\theta)
     / (2 + \theta)} | \log t|^{4 / (2 + \theta)} + \lambda^{6}) Q_{\rho}
     (\mathbb{X}_{\varepsilon})]^{(p + 1) / 2} .
     \end{equation}
\end{corollary}

\begin{proof}
  Based on {\eqref{eq:d18}} we obtain
  \[ \frac{1}{2 p} \partial_t \| \rho^2 \phi_{\varepsilon} \|_{L^{2,
     \varepsilon}}^{2 p} + \lambda \| \rho^2 \phi_{\varepsilon} \|_{L^{2,
     \varepsilon}}^{2 (p - 1)} \| \rho \phi_{\varepsilon} \|_{L^{4,
     \varepsilon}}^4 
     \]
     \[
     \leqslant (\lambda^3 + \lambda^{(12 - \theta)
     / (2 + \theta)} | \log t|^{4 / (2 + \theta)} + \lambda^{7}) \| \rho^2 \phi_{\varepsilon}
     \|_{L^{2, \varepsilon}}^{2 (p - 1)} Q_{\rho} (\mathbb{X}_{\varepsilon}) .
  \]
  The $L^4$-norm on the left hand side can be estimated from below by the
  $L^2$-norm, whereas on the right hand side we use Young's inequality to
  deduce
  \[ \frac{1}{2 p} \partial_t \| \rho^2 \phi_{\varepsilon} \|_{L^{2,
     \varepsilon}}^{2 p} + \lambda \| \rho^2 \phi_{\varepsilon} \|_{L^{2,
     \varepsilon}}^{2 p + 2}
     \]
     \[
      \leqslant \lambda [(\lambda^{2} + \lambda^{(10 - 2\theta)
     / (2 + \theta)} | \log t|^{4 / (2 + \theta)} + \lambda^{6}) Q_{\rho}
     (\mathbb{X}_{\varepsilon})]^{(p + 1) / 2} + \delta \lambda \| \rho^2
     \phi_{\varepsilon} \|_{L^{2, \varepsilon}}^{2 p + 2} .\]
Hence we may absorb the second term from the right hand side into the left hand
  side.
\end{proof}

\subsection{Tightness of the invariant measures}

\label{ss:tight}Recall that $\varphi_{M, \varepsilon}$ is a stationary
solution to {\eqref{eq:moll}} having at time $t \geqslant 0$ law given by the
Gibbs measure $\nu_{M, \varepsilon}$. Moreover, we have the decomposition
$\varphi_{M, \varepsilon} = X_{M, \varepsilon} + Y_{M, \varepsilon} + \phi_{M,
\varepsilon}$, where $X_{M, \varepsilon}$ is stationary as well. By our construction, all equations are solved on a common
probability space, say $(\Omega, \mathcal{F}, \mathbb{P})$, and we denote by
$\mathbb{E}$ the corresponding expected value. In addition, we assume that the processes $\varphi_{M,\varepsilon}$ and $X_{M,\varepsilon}$ are  jointly stationary. This could be achieved for instance by considering a solution to the coupled SDE for $(\varphi_{M,\varepsilon},X_{M,\varepsilon})$ starting from the product of the corresponding marginal invariant measures, and applying Krylov--Bogoliubov's argument.

\begin{theorem}
  \label{thm:tight}Let $\rho$ be a weight such that $\rho^{\iota} \in L^{4,
  0}$ for some $\iota \in (0, 1)$. Then for every $p \in [1, \infty)$
  \[ \sup_{\varepsilon \in \mathcal{A}, M > 0} (\mathbb{E} \| \varphi_{M,
     \varepsilon} (0) - X_{M, \varepsilon} (0) \|_{H^{1 / 2 - 2 \kappa,
     \varepsilon} (\rho^2)}^2)^{1/2}\lesssim {\lambda} +  \lambda^{7/2},\]
  \[    \sup_{\varepsilon \in \mathcal{A}, M > 0}   (\mathbb{E} \| \varphi_{M,
     \varepsilon} (0) - X_{M, \varepsilon} (0) \|_{L^{2, \varepsilon}
     (\rho^2)}^{2 p})^{1/2p}\lesssim {\lambda^{1/2}} +  {\lambda^{3/2}}.\]
\end{theorem}

\begin{proof}
  Let us  show the first claim. Due to stationarity of $\varphi_{M,
  \varepsilon} - X_{M, \varepsilon} = Y_{M, \varepsilon} + \phi_{M,
  \varepsilon}$ we obtain
  \[ \mathbb{E} \| \rho^2 (\varphi_{M, \varepsilon} (0) - X_{M, \varepsilon}
     (0)) \|_{H^{1 / 2 - 2 \kappa, \varepsilon}}^2 = \frac{1}{\tau}
     \int_0^{\tau} \mathbb{E} \| \rho^2 (\varphi_{M, \varepsilon} (s) - X_{M,
     \varepsilon} (s)) \|_{H^{1 / 2 - 2 \kappa, \varepsilon}}^2 \mathd s \]
  \[ = \frac{1}{\tau} \int_0^{\tau} \mathbb{E} \| \rho^2 (\phi_{M,
     \varepsilon} (s) + Y_{M, \varepsilon} (s)) \|_{H^{1 / 2 - 2 \kappa,
     \varepsilon}}^2 \mathd s \]
  \[ \lesssim \frac{1}{\tau} \int_0^{\tau} \mathbb{E} \| \rho^2 \phi_{M,
     \varepsilon} (s) \|_{H^{1 / 2 - 2 \kappa, \varepsilon}}^2 \mathd s +
     \frac{1}{\tau} \int_0^{\tau} \mathbb{E} \| \rho^2 Y_{M, \varepsilon} (s)
     \|_{H^{1 / 2 - 2 \kappa, \varepsilon}}^2 \mathd s. \]
  In order to estimate the  right hand side, we employ Theorem \ref{th:energy-estimate} together with Lemma~\ref{lem:Y1} to deduce
  \[ \mathbb{E} \| \rho^2 (\varphi_{M, \varepsilon} (0) - X_{M, \varepsilon}
     (0)) \|_{H^{1 / 2 - 2 \kappa, \varepsilon}}^2 \]
  \[ \lesssim C_{\tau}(\lambda^2 +  \lambda^{7})\mathbb{E}Q_{\rho}
     (\mathbb{X}_{M, \varepsilon}) + \frac{1}{2\tau} \mathbb{E} \| \rho^2 \phi_{M,
     \varepsilon} (0) \|_{L^{2, \varepsilon}}^2 +\mathbb{E}\|\rho^{\sigma}Y_{M,\varepsilon}\|_{C_{T}\CC^{1/2-\kappa,\varepsilon}}^{2}\]
  \[ \leqslant C_{\tau} (\lambda^2 +  \lambda^{7})\mathbb{E}Q_{\rho} (\mathbb{X}_{M, \varepsilon}) +\frac{C}{\tau}
     \mathbb{E} \| \rho^2 (\varphi_{M, \varepsilon} (0) - X_{M, \varepsilon}
     (0)) \|_{L^{2, \varepsilon}}^2 + \frac{C}{\tau}\mathbb{E} \| \rho^2 Y_{M,\varepsilon}
     (0) \|_{L^{2, \varepsilon}}^2 \]
  \[ \leqslant C_{\tau}(\lambda^2 +  \lambda^{7}) \mathbb{E}Q_{\rho} (\mathbb{X}_{M, \varepsilon}) +
     \frac{C}{\tau}\mathbb{E} \| \rho^2 (\varphi_{M, \varepsilon} (0) - X_{M, \varepsilon}
     (0)) \|_{L^{2, \varepsilon}}^2 .\]
  Finally, taking $\tau > 0$ large enough, we may absorb the second term from
  the right hand side into the left hand side to deduce
  \[ \mathbb{E} \| \rho^2 (\varphi_{M, \varepsilon} (0) - X_{M, \varepsilon}
     (0)) \|_{H^{1 / 2 - 2 \kappa, \varepsilon}}^2 \leqslant C_{\tau} ({\lambda^2} +  \lambda^{7})
     \mathbb{E}Q_{\rho} (\mathbb{X}_{M, \varepsilon}) . \]
  Observing that the right hand side is bounded uniformly in $M,\varepsilon$, completes the proof of the first claim.
  
  Now, we show the second claim for $p \in [2, \infty)$. The case $p \in [1,
  2)$ then follows easily from the bound for $p=2$. Using stationarity as above we have
  \[ \mathbb{E} \| \rho^2 (\varphi_{M, \varepsilon} (0) - X_{M, \varepsilon}
     (0)) \|_{L^{2, \varepsilon}}^{2 p} = \frac{1}{\tau} \int_0^{\tau}
     \mathbb{E} \| \rho^2 (\phi_{M, \varepsilon} (s) + Y_{M, \varepsilon} (s))
     \|_{L^{2, \varepsilon}}^{2 p} \mathd s \]
  \[ \  \]
  \begin{equation}
    \lesssim \frac{1}{\tau} \int_0^{\tau} \mathbb{E} \| \rho^2 \phi_{M,
    \varepsilon} (s) \|_{L^{2, \varepsilon}}^{2 p} \mathd s + \frac{1}{\tau}
    \int_0^{\tau} \mathbb{E} \| \rho^2 Y_{M, \varepsilon} (s) \|_{L^{2,
    \varepsilon}}^{2 p} \mathd s. \label{eq:d21}
  \end{equation}
  Due to Corollary~\ref{cor:Lp} applied to $p - 1$ and the fact that for any $\sigma>0$ and  $\tau\geqslant 1$
  $$\int_{0}^{\tau}| \log s|^{2p/(2+\theta)}\mathrm{d} s\leqslant C_{p,\sigma}\tau^{1+\sigma},$$
  we deduce
  \[ \begin{aligned}
  \alpha \int_0^{\tau} \mathbb{E} \| \rho^2 \phi_{M, \varepsilon} (s)
     \|_{L^{2, \varepsilon}}^{2 p} \mathd s & \leqslant C_{p,\sigma} [\tau(\lambda^{2} + \lambda^{6})^{p/2}+\tau^{1+\sigma}\lambda^{p(5-\theta)/(2+\theta)}] \mathbb{E}
     [Q_{\rho} (\mathbb{X}_{M,\varepsilon})] \\
     &\qquad + \frac{\lambda^{-1}}{2 (p - 1)}
     \mathbb{E} \| \rho^2 \phi_{M, \varepsilon} (0) \|_{L^{2, \varepsilon}}^{2
     (p - 1)} \\
  & \leqslant C_{p,\sigma}[\tau(\lambda^{2} + \lambda^{6})^{p/2}+\tau^{1+\sigma}\lambda^{p(5-\theta)/(2+\theta)}]  \mathbb{E} [Q_{\rho} (\mathbb{X}_{M,\varepsilon})] 
     \\ & \qquad + C_p  \lambda^{-1}\mathbb{E} \| \rho^2 (\varphi_{M, \varepsilon} (0) - X_{M,
     \varepsilon} (0)) \|_{L^{2, \varepsilon}}^{2 (p - 1)} 
     \\ & \qquad + C_p \lambda^{-1}\mathbb{E} \|
     \rho^2 Y_{M, \varepsilon} (0) \|_{L^{2, \varepsilon}}^{2 (p - 1)} . \end{aligned}\]
  Plugging this back into {\eqref{eq:d21}} and using Young's inequality we
  obtain
  \[ \mathbb{E} \| \rho^2 (\varphi_{M, \varepsilon} (0) - X_{M, \varepsilon}
     (0)) \|_{L^{2, \varepsilon}}^{2 p} \leqslant \frac{C_{p,\sigma}}{\alpha}[(\lambda^{2} + \lambda^{6})^{p/2}+\tau^{\sigma}\lambda^{p(5-\theta)/(2+\theta)}] \mathbb{E}
     [Q_{\rho} (\mathbb{X}_{M,\varepsilon})]  \]
  \[ + \delta \frac{C_p}{\alpha
     \tau}  \mathbb{E} \| \rho^2 (\varphi_{M, \varepsilon} (0) -
     X_{M, \varepsilon} (0)) \|_{L^{2, \varepsilon}}^{2 p} + \frac{1}{\lambda^p \tau}C_{\delta, p} {+\frac{C_{p}\lambda^{2p}}{\alpha\tau}\mathbb{E}[Q_{\rho}(\mathbb{X}_{M,\varepsilon})]}. \]
Taking $\tau = {\max}(1,\lambda^{{-2p}})$ leads to
 \[ \mathbb{E} \| \rho^2 (\varphi_{M, \varepsilon} (0) - X_{M, \varepsilon}
     (0)) \|_{L^{2, \varepsilon}}^{2 p} \leqslant \frac{C_{p,\sigma}}{\alpha}[(\lambda^{2}+ \lambda^{6})^{p/2}+\tau^{\sigma}\lambda^{p(5-\theta)/(2+\theta)}] \mathbb{E}
     [Q_{\rho} (\mathbb{X}_{M,\varepsilon})]  \]
  \[ + \delta C_{p,\alpha}  \mathbb{E} \| \rho^2 (\varphi_{M, \varepsilon} (0) -
     X_{M, \varepsilon} (0)) \|_{L^{2, \varepsilon}}^{2 p} + {\lambda^{p}}C_{\delta, p} {+C_{p,\alpha}\lambda^{2p}\mathbb{E}[Q_{\rho}(\mathbb{X}_{M,\varepsilon})]} \]
and choosing $\delta>0$ small enough, we may absorb the second term on the right hand side into the left hand side and the claim follows
\end{proof}

The above result directly implies the desired tightness of the approximate
Gibbs measures $\nu_{M, \varepsilon}$. To formulate this precisely we make use
of the extension operators $\mathcal{E}^{\varepsilon}$ for distributions on
$\Lambda_{\varepsilon}$ constructed in Section~\ref{s:ext}.
We recall that on
the approximate level the stationary process $\varphi_{M, \varepsilon}$ admits the decomposition
$ \varphi_{M, \varepsilon} = X_{M, \varepsilon} + Y_{M, \varepsilon} +
   \phi_{M, \varepsilon}, $
where $X_{M,\varepsilon}$ is stationary and $Y_{M,\varepsilon}$ is given by \eqref{eq:YY} with $X_{M, \varepsilon}^{\!\resizebox{0.6em}{!}{
\begin{tikzpicture}
\pgfpathmoveto{\pgfqpoint{0cm}{-0.035cm}}
\pgfpathlineto{\pgfqpoint{1.376cm}{-0.035cm}}
\pgfpathlineto{\pgfqpoint{1.376cm}{1.552cm}}
\pgfpathlineto{\pgfqpoint{0cm}{1.552cm}}
\pgfpathclose
\pgfusepath{clip}
\begin{pgfscope}
\begin{pgfscope}
\pgfpathmoveto{\pgfqpoint{0cm}{-0.035cm}}
\pgfpathlineto{\pgfqpoint{1.376cm}{-0.035cm}}
\pgfpathlineto{\pgfqpoint{1.376cm}{1.552cm}}
\pgfpathlineto{\pgfqpoint{0cm}{1.552cm}}
\pgfpathclose
\pgfusepath{clip}
\begin{pgfscope}
\begin{pgfscope}
\pgfsetdash{}{0cm}
\pgfsetlinewidth{0.818mm}
\pgfsetroundcap
\pgfsetroundjoin
\pgfsetmiterlimit{7.0}
\definecolor{eps2pgf_color}{gray}{0}\pgfsetstrokecolor{eps2pgf_color}\pgfsetfillcolor{eps2pgf_color}
\pgfpathmoveto{\pgfqpoint{0.117cm}{1.421cm}}
\pgfpathlineto{\pgfqpoint{0.682cm}{0.671cm}}
\pgfpathlineto{\pgfqpoint{1.246cm}{1.421cm}}
\pgfusepath{stroke}
\end{pgfscope}
\definecolor{eps2pgf_color}{gray}{0}\pgfsetstrokecolor{eps2pgf_color}\pgfsetfillcolor{eps2pgf_color}
\pgfpathmoveto{\pgfqpoint{0.273cm}{1.395cm}}
\pgfpathcurveto{\pgfqpoint{0.273cm}{1.432cm}}{\pgfqpoint{0.259cm}{1.467cm}}{\pgfqpoint{0.233cm}{1.492cm}}
\pgfpathcurveto{\pgfqpoint{0.207cm}{1.518cm}}{\pgfqpoint{0.173cm}{1.532cm}}{\pgfqpoint{0.137cm}{1.532cm}}
\pgfpathcurveto{\pgfqpoint{0.1cm}{1.532cm}}{\pgfqpoint{0.066cm}{1.518cm}}{\pgfqpoint{0.04cm}{1.492cm}}
\pgfpathcurveto{\pgfqpoint{0.014cm}{1.467cm}}{\pgfqpoint{0cm}{1.432cm}}{\pgfqpoint{0cm}{1.395cm}}
\pgfpathcurveto{\pgfqpoint{0cm}{1.359cm}}{\pgfqpoint{0.014cm}{1.324cm}}{\pgfqpoint{0.04cm}{1.299cm}}
\pgfpathcurveto{\pgfqpoint{0.066cm}{1.273cm}}{\pgfqpoint{0.1cm}{1.258cm}}{\pgfqpoint{0.137cm}{1.258cm}}
\pgfpathcurveto{\pgfqpoint{0.173cm}{1.258cm}}{\pgfqpoint{0.207cm}{1.273cm}}{\pgfqpoint{0.233cm}{1.299cm}}
\pgfpathcurveto{\pgfqpoint{0.259cm}{1.324cm}}{\pgfqpoint{0.273cm}{1.359cm}}{\pgfqpoint{0.273cm}{1.395cm}}
\pgfusepath{fill}
\begin{pgfscope}
\pgfsetdash{}{0cm}
\pgfsetlinewidth{0.818mm}
\pgfsetmiterlimit{7.0}
\pgfpathmoveto{\pgfqpoint{0.682cm}{0.671cm}}
\pgfpathlineto{\pgfqpoint{0.679cm}{1.418cm}}
\pgfusepath{stroke}
\end{pgfscope}
\pgfpathmoveto{\pgfqpoint{0.815cm}{1.399cm}}
\pgfpathcurveto{\pgfqpoint{0.815cm}{1.435cm}}{\pgfqpoint{0.801cm}{1.47cm}}{\pgfqpoint{0.775cm}{1.496cm}}
\pgfpathcurveto{\pgfqpoint{0.75cm}{1.521cm}}{\pgfqpoint{0.715cm}{1.536cm}}{\pgfqpoint{0.679cm}{1.536cm}}
\pgfpathcurveto{\pgfqpoint{0.643cm}{1.536cm}}{\pgfqpoint{0.608cm}{1.521cm}}{\pgfqpoint{0.582cm}{1.496cm}}
\pgfpathcurveto{\pgfqpoint{0.557cm}{1.47cm}}{\pgfqpoint{0.542cm}{1.435cm}}{\pgfqpoint{0.542cm}{1.399cm}}
\pgfpathcurveto{\pgfqpoint{0.542cm}{1.363cm}}{\pgfqpoint{0.557cm}{1.328cm}}{\pgfqpoint{0.582cm}{1.302cm}}
\pgfpathcurveto{\pgfqpoint{0.608cm}{1.276cm}}{\pgfqpoint{0.643cm}{1.262cm}}{\pgfqpoint{0.679cm}{1.262cm}}
\pgfpathcurveto{\pgfqpoint{0.715cm}{1.262cm}}{\pgfqpoint{0.75cm}{1.276cm}}{\pgfqpoint{0.775cm}{1.302cm}}
\pgfpathcurveto{\pgfqpoint{0.801cm}{1.328cm}}{\pgfqpoint{0.815cm}{1.363cm}}{\pgfqpoint{0.815cm}{1.399cm}}
\pgfusepath{fill}
\pgfpathmoveto{\pgfqpoint{1.345cm}{1.371cm}}
\pgfpathcurveto{\pgfqpoint{1.345cm}{1.408cm}}{\pgfqpoint{1.331cm}{1.442cm}}{\pgfqpoint{1.305cm}{1.468cm}}
\pgfpathcurveto{\pgfqpoint{1.28cm}{1.494cm}}{\pgfqpoint{1.245cm}{1.508cm}}{\pgfqpoint{1.209cm}{1.508cm}}
\pgfpathcurveto{\pgfqpoint{1.172cm}{1.508cm}}{\pgfqpoint{1.138cm}{1.494cm}}{\pgfqpoint{1.112cm}{1.468cm}}
\pgfpathcurveto{\pgfqpoint{1.087cm}{1.442cm}}{\pgfqpoint{1.072cm}{1.408cm}}{\pgfqpoint{1.072cm}{1.371cm}}
\pgfpathcurveto{\pgfqpoint{1.072cm}{1.335cm}}{\pgfqpoint{1.087cm}{1.3cm}}{\pgfqpoint{1.112cm}{1.274cm}}
\pgfpathcurveto{\pgfqpoint{1.138cm}{1.249cm}}{\pgfqpoint{1.172cm}{1.234cm}}{\pgfqpoint{1.209cm}{1.234cm}}
\pgfpathcurveto{\pgfqpoint{1.245cm}{1.234cm}}{\pgfqpoint{1.28cm}{1.249cm}}{\pgfqpoint{1.305cm}{1.274cm}}
\pgfpathcurveto{\pgfqpoint{1.331cm}{1.3cm}}{\pgfqpoint{1.345cm}{1.335cm}}{\pgfqpoint{1.345cm}{1.371cm}}
\pgfusepath{fill}
\begin{pgfscope}
\pgfsetdash{}{0cm}
\pgfsetlinewidth{0.818mm}
\pgfsetroundcap
\pgfsetmiterlimit{4.0}
\pgfpathmoveto{\pgfqpoint{0.682cm}{0.671cm}}
\pgfpathlineto{\pgfqpoint{0.682cm}{0.042cm}}
\pgfusepath{stroke}
\end{pgfscope}
\end{pgfscope}
\end{pgfscope}
\end{pgfscope}
\end{tikzpicture}}}$ being also stationary.
Accordingly, letting
\[ \zeta_{M, \varepsilon} \assign - \LL_{\varepsilon}^{- 1} \left[ 3\lambda \left(
   \UU^{\varepsilon}_{>} \llbracket X_{M, \varepsilon}^2 \rrbracket \right)
   \succ Y_{M, \varepsilon} \right] + \phi_{M, \varepsilon} \backassign
   \eta_{M, \varepsilon} + \phi_{M, \varepsilon} \]
we obtain $\varphi_{M, \varepsilon} = X_{M, \varepsilon} - \lambda X_{M,
\varepsilon}^{\!\resizebox{0.6em}{!}{
\begin{tikzpicture}
\pgfpathmoveto{\pgfqpoint{0cm}{-0.035cm}}
\pgfpathlineto{\pgfqpoint{1.376cm}{-0.035cm}}
\pgfpathlineto{\pgfqpoint{1.376cm}{1.552cm}}
\pgfpathlineto{\pgfqpoint{0cm}{1.552cm}}
\pgfpathclose
\pgfusepath{clip}
\begin{pgfscope}
\begin{pgfscope}
\pgfpathmoveto{\pgfqpoint{0cm}{-0.035cm}}
\pgfpathlineto{\pgfqpoint{1.376cm}{-0.035cm}}
\pgfpathlineto{\pgfqpoint{1.376cm}{1.552cm}}
\pgfpathlineto{\pgfqpoint{0cm}{1.552cm}}
\pgfpathclose
\pgfusepath{clip}
\begin{pgfscope}
\begin{pgfscope}
\pgfsetdash{}{0cm}
\pgfsetlinewidth{0.818mm}
\pgfsetroundcap
\pgfsetroundjoin
\pgfsetmiterlimit{7.0}
\definecolor{eps2pgf_color}{gray}{0}\pgfsetstrokecolor{eps2pgf_color}\pgfsetfillcolor{eps2pgf_color}
\pgfpathmoveto{\pgfqpoint{0.117cm}{1.421cm}}
\pgfpathlineto{\pgfqpoint{0.682cm}{0.671cm}}
\pgfpathlineto{\pgfqpoint{1.246cm}{1.421cm}}
\pgfusepath{stroke}
\end{pgfscope}
\definecolor{eps2pgf_color}{gray}{0}\pgfsetstrokecolor{eps2pgf_color}\pgfsetfillcolor{eps2pgf_color}
\pgfpathmoveto{\pgfqpoint{0.273cm}{1.395cm}}
\pgfpathcurveto{\pgfqpoint{0.273cm}{1.432cm}}{\pgfqpoint{0.259cm}{1.467cm}}{\pgfqpoint{0.233cm}{1.492cm}}
\pgfpathcurveto{\pgfqpoint{0.207cm}{1.518cm}}{\pgfqpoint{0.173cm}{1.532cm}}{\pgfqpoint{0.137cm}{1.532cm}}
\pgfpathcurveto{\pgfqpoint{0.1cm}{1.532cm}}{\pgfqpoint{0.066cm}{1.518cm}}{\pgfqpoint{0.04cm}{1.492cm}}
\pgfpathcurveto{\pgfqpoint{0.014cm}{1.467cm}}{\pgfqpoint{0cm}{1.432cm}}{\pgfqpoint{0cm}{1.395cm}}
\pgfpathcurveto{\pgfqpoint{0cm}{1.359cm}}{\pgfqpoint{0.014cm}{1.324cm}}{\pgfqpoint{0.04cm}{1.299cm}}
\pgfpathcurveto{\pgfqpoint{0.066cm}{1.273cm}}{\pgfqpoint{0.1cm}{1.258cm}}{\pgfqpoint{0.137cm}{1.258cm}}
\pgfpathcurveto{\pgfqpoint{0.173cm}{1.258cm}}{\pgfqpoint{0.207cm}{1.273cm}}{\pgfqpoint{0.233cm}{1.299cm}}
\pgfpathcurveto{\pgfqpoint{0.259cm}{1.324cm}}{\pgfqpoint{0.273cm}{1.359cm}}{\pgfqpoint{0.273cm}{1.395cm}}
\pgfusepath{fill}
\begin{pgfscope}
\pgfsetdash{}{0cm}
\pgfsetlinewidth{0.818mm}
\pgfsetmiterlimit{7.0}
\pgfpathmoveto{\pgfqpoint{0.682cm}{0.671cm}}
\pgfpathlineto{\pgfqpoint{0.679cm}{1.418cm}}
\pgfusepath{stroke}
\end{pgfscope}
\pgfpathmoveto{\pgfqpoint{0.815cm}{1.399cm}}
\pgfpathcurveto{\pgfqpoint{0.815cm}{1.435cm}}{\pgfqpoint{0.801cm}{1.47cm}}{\pgfqpoint{0.775cm}{1.496cm}}
\pgfpathcurveto{\pgfqpoint{0.75cm}{1.521cm}}{\pgfqpoint{0.715cm}{1.536cm}}{\pgfqpoint{0.679cm}{1.536cm}}
\pgfpathcurveto{\pgfqpoint{0.643cm}{1.536cm}}{\pgfqpoint{0.608cm}{1.521cm}}{\pgfqpoint{0.582cm}{1.496cm}}
\pgfpathcurveto{\pgfqpoint{0.557cm}{1.47cm}}{\pgfqpoint{0.542cm}{1.435cm}}{\pgfqpoint{0.542cm}{1.399cm}}
\pgfpathcurveto{\pgfqpoint{0.542cm}{1.363cm}}{\pgfqpoint{0.557cm}{1.328cm}}{\pgfqpoint{0.582cm}{1.302cm}}
\pgfpathcurveto{\pgfqpoint{0.608cm}{1.276cm}}{\pgfqpoint{0.643cm}{1.262cm}}{\pgfqpoint{0.679cm}{1.262cm}}
\pgfpathcurveto{\pgfqpoint{0.715cm}{1.262cm}}{\pgfqpoint{0.75cm}{1.276cm}}{\pgfqpoint{0.775cm}{1.302cm}}
\pgfpathcurveto{\pgfqpoint{0.801cm}{1.328cm}}{\pgfqpoint{0.815cm}{1.363cm}}{\pgfqpoint{0.815cm}{1.399cm}}
\pgfusepath{fill}
\pgfpathmoveto{\pgfqpoint{1.345cm}{1.371cm}}
\pgfpathcurveto{\pgfqpoint{1.345cm}{1.408cm}}{\pgfqpoint{1.331cm}{1.442cm}}{\pgfqpoint{1.305cm}{1.468cm}}
\pgfpathcurveto{\pgfqpoint{1.28cm}{1.494cm}}{\pgfqpoint{1.245cm}{1.508cm}}{\pgfqpoint{1.209cm}{1.508cm}}
\pgfpathcurveto{\pgfqpoint{1.172cm}{1.508cm}}{\pgfqpoint{1.138cm}{1.494cm}}{\pgfqpoint{1.112cm}{1.468cm}}
\pgfpathcurveto{\pgfqpoint{1.087cm}{1.442cm}}{\pgfqpoint{1.072cm}{1.408cm}}{\pgfqpoint{1.072cm}{1.371cm}}
\pgfpathcurveto{\pgfqpoint{1.072cm}{1.335cm}}{\pgfqpoint{1.087cm}{1.3cm}}{\pgfqpoint{1.112cm}{1.274cm}}
\pgfpathcurveto{\pgfqpoint{1.138cm}{1.249cm}}{\pgfqpoint{1.172cm}{1.234cm}}{\pgfqpoint{1.209cm}{1.234cm}}
\pgfpathcurveto{\pgfqpoint{1.245cm}{1.234cm}}{\pgfqpoint{1.28cm}{1.249cm}}{\pgfqpoint{1.305cm}{1.274cm}}
\pgfpathcurveto{\pgfqpoint{1.331cm}{1.3cm}}{\pgfqpoint{1.345cm}{1.335cm}}{\pgfqpoint{1.345cm}{1.371cm}}
\pgfusepath{fill}
\begin{pgfscope}
\pgfsetdash{}{0cm}
\pgfsetlinewidth{0.818mm}
\pgfsetroundcap
\pgfsetmiterlimit{4.0}
\pgfpathmoveto{\pgfqpoint{0.682cm}{0.671cm}}
\pgfpathlineto{\pgfqpoint{0.682cm}{0.042cm}}
\pgfusepath{stroke}
\end{pgfscope}
\end{pgfscope}
\end{pgfscope}
\end{pgfscope}
\end{tikzpicture}}} + \zeta_{M, \varepsilon}$, where all the summands are
stationary.

The next result shows that the family of joint laws of $( \mathcal{E}^{\varepsilon} \varphi_{M, \varepsilon},
\mathcal{E}^{\varepsilon} X_{M, \varepsilon},\mathcal{E}^{\varepsilon} X^{\!\resizebox{0.6em}{!}{
\begin{tikzpicture}
\pgfpathmoveto{\pgfqpoint{0cm}{-0.035cm}}
\pgfpathlineto{\pgfqpoint{1.376cm}{-0.035cm}}
\pgfpathlineto{\pgfqpoint{1.376cm}{1.552cm}}
\pgfpathlineto{\pgfqpoint{0cm}{1.552cm}}
\pgfpathclose
\pgfusepath{clip}
\begin{pgfscope}
\begin{pgfscope}
\pgfpathmoveto{\pgfqpoint{0cm}{-0.035cm}}
\pgfpathlineto{\pgfqpoint{1.376cm}{-0.035cm}}
\pgfpathlineto{\pgfqpoint{1.376cm}{1.552cm}}
\pgfpathlineto{\pgfqpoint{0cm}{1.552cm}}
\pgfpathclose
\pgfusepath{clip}
\begin{pgfscope}
\begin{pgfscope}
\pgfsetdash{}{0cm}
\pgfsetlinewidth{0.818mm}
\pgfsetroundcap
\pgfsetroundjoin
\pgfsetmiterlimit{7.0}
\definecolor{eps2pgf_color}{gray}{0}\pgfsetstrokecolor{eps2pgf_color}\pgfsetfillcolor{eps2pgf_color}
\pgfpathmoveto{\pgfqpoint{0.117cm}{1.421cm}}
\pgfpathlineto{\pgfqpoint{0.682cm}{0.671cm}}
\pgfpathlineto{\pgfqpoint{1.246cm}{1.421cm}}
\pgfusepath{stroke}
\end{pgfscope}
\definecolor{eps2pgf_color}{gray}{0}\pgfsetstrokecolor{eps2pgf_color}\pgfsetfillcolor{eps2pgf_color}
\pgfpathmoveto{\pgfqpoint{0.273cm}{1.395cm}}
\pgfpathcurveto{\pgfqpoint{0.273cm}{1.432cm}}{\pgfqpoint{0.259cm}{1.467cm}}{\pgfqpoint{0.233cm}{1.492cm}}
\pgfpathcurveto{\pgfqpoint{0.207cm}{1.518cm}}{\pgfqpoint{0.173cm}{1.532cm}}{\pgfqpoint{0.137cm}{1.532cm}}
\pgfpathcurveto{\pgfqpoint{0.1cm}{1.532cm}}{\pgfqpoint{0.066cm}{1.518cm}}{\pgfqpoint{0.04cm}{1.492cm}}
\pgfpathcurveto{\pgfqpoint{0.014cm}{1.467cm}}{\pgfqpoint{0cm}{1.432cm}}{\pgfqpoint{0cm}{1.395cm}}
\pgfpathcurveto{\pgfqpoint{0cm}{1.359cm}}{\pgfqpoint{0.014cm}{1.324cm}}{\pgfqpoint{0.04cm}{1.299cm}}
\pgfpathcurveto{\pgfqpoint{0.066cm}{1.273cm}}{\pgfqpoint{0.1cm}{1.258cm}}{\pgfqpoint{0.137cm}{1.258cm}}
\pgfpathcurveto{\pgfqpoint{0.173cm}{1.258cm}}{\pgfqpoint{0.207cm}{1.273cm}}{\pgfqpoint{0.233cm}{1.299cm}}
\pgfpathcurveto{\pgfqpoint{0.259cm}{1.324cm}}{\pgfqpoint{0.273cm}{1.359cm}}{\pgfqpoint{0.273cm}{1.395cm}}
\pgfusepath{fill}
\begin{pgfscope}
\pgfsetdash{}{0cm}
\pgfsetlinewidth{0.818mm}
\pgfsetmiterlimit{7.0}
\pgfpathmoveto{\pgfqpoint{0.682cm}{0.671cm}}
\pgfpathlineto{\pgfqpoint{0.679cm}{1.418cm}}
\pgfusepath{stroke}
\end{pgfscope}
\pgfpathmoveto{\pgfqpoint{0.815cm}{1.399cm}}
\pgfpathcurveto{\pgfqpoint{0.815cm}{1.435cm}}{\pgfqpoint{0.801cm}{1.47cm}}{\pgfqpoint{0.775cm}{1.496cm}}
\pgfpathcurveto{\pgfqpoint{0.75cm}{1.521cm}}{\pgfqpoint{0.715cm}{1.536cm}}{\pgfqpoint{0.679cm}{1.536cm}}
\pgfpathcurveto{\pgfqpoint{0.643cm}{1.536cm}}{\pgfqpoint{0.608cm}{1.521cm}}{\pgfqpoint{0.582cm}{1.496cm}}
\pgfpathcurveto{\pgfqpoint{0.557cm}{1.47cm}}{\pgfqpoint{0.542cm}{1.435cm}}{\pgfqpoint{0.542cm}{1.399cm}}
\pgfpathcurveto{\pgfqpoint{0.542cm}{1.363cm}}{\pgfqpoint{0.557cm}{1.328cm}}{\pgfqpoint{0.582cm}{1.302cm}}
\pgfpathcurveto{\pgfqpoint{0.608cm}{1.276cm}}{\pgfqpoint{0.643cm}{1.262cm}}{\pgfqpoint{0.679cm}{1.262cm}}
\pgfpathcurveto{\pgfqpoint{0.715cm}{1.262cm}}{\pgfqpoint{0.75cm}{1.276cm}}{\pgfqpoint{0.775cm}{1.302cm}}
\pgfpathcurveto{\pgfqpoint{0.801cm}{1.328cm}}{\pgfqpoint{0.815cm}{1.363cm}}{\pgfqpoint{0.815cm}{1.399cm}}
\pgfusepath{fill}
\pgfpathmoveto{\pgfqpoint{1.345cm}{1.371cm}}
\pgfpathcurveto{\pgfqpoint{1.345cm}{1.408cm}}{\pgfqpoint{1.331cm}{1.442cm}}{\pgfqpoint{1.305cm}{1.468cm}}
\pgfpathcurveto{\pgfqpoint{1.28cm}{1.494cm}}{\pgfqpoint{1.245cm}{1.508cm}}{\pgfqpoint{1.209cm}{1.508cm}}
\pgfpathcurveto{\pgfqpoint{1.172cm}{1.508cm}}{\pgfqpoint{1.138cm}{1.494cm}}{\pgfqpoint{1.112cm}{1.468cm}}
\pgfpathcurveto{\pgfqpoint{1.087cm}{1.442cm}}{\pgfqpoint{1.072cm}{1.408cm}}{\pgfqpoint{1.072cm}{1.371cm}}
\pgfpathcurveto{\pgfqpoint{1.072cm}{1.335cm}}{\pgfqpoint{1.087cm}{1.3cm}}{\pgfqpoint{1.112cm}{1.274cm}}
\pgfpathcurveto{\pgfqpoint{1.138cm}{1.249cm}}{\pgfqpoint{1.172cm}{1.234cm}}{\pgfqpoint{1.209cm}{1.234cm}}
\pgfpathcurveto{\pgfqpoint{1.245cm}{1.234cm}}{\pgfqpoint{1.28cm}{1.249cm}}{\pgfqpoint{1.305cm}{1.274cm}}
\pgfpathcurveto{\pgfqpoint{1.331cm}{1.3cm}}{\pgfqpoint{1.345cm}{1.335cm}}{\pgfqpoint{1.345cm}{1.371cm}}
\pgfusepath{fill}
\begin{pgfscope}
\pgfsetdash{}{0cm}
\pgfsetlinewidth{0.818mm}
\pgfsetroundcap
\pgfsetmiterlimit{4.0}
\pgfpathmoveto{\pgfqpoint{0.682cm}{0.671cm}}
\pgfpathlineto{\pgfqpoint{0.682cm}{0.042cm}}
\pgfusepath{stroke}
\end{pgfscope}
\end{pgfscope}
\end{pgfscope}
\end{pgfscope}
\end{tikzpicture}}}_{M, \varepsilon})$ at any chosen time $t\geqslant 0$ is tight. In addition, we obtain bounds for arbitrary moments of the limiting measure. To this end, we denote by $(\varphi,X, X^{\!\resizebox{0.6em}{!}{
\begin{tikzpicture}
\pgfpathmoveto{\pgfqpoint{0cm}{-0.035cm}}
\pgfpathlineto{\pgfqpoint{1.376cm}{-0.035cm}}
\pgfpathlineto{\pgfqpoint{1.376cm}{1.552cm}}
\pgfpathlineto{\pgfqpoint{0cm}{1.552cm}}
\pgfpathclose
\pgfusepath{clip}
\begin{pgfscope}
\begin{pgfscope}
\pgfpathmoveto{\pgfqpoint{0cm}{-0.035cm}}
\pgfpathlineto{\pgfqpoint{1.376cm}{-0.035cm}}
\pgfpathlineto{\pgfqpoint{1.376cm}{1.552cm}}
\pgfpathlineto{\pgfqpoint{0cm}{1.552cm}}
\pgfpathclose
\pgfusepath{clip}
\begin{pgfscope}
\begin{pgfscope}
\pgfsetdash{}{0cm}
\pgfsetlinewidth{0.818mm}
\pgfsetroundcap
\pgfsetroundjoin
\pgfsetmiterlimit{7.0}
\definecolor{eps2pgf_color}{gray}{0}\pgfsetstrokecolor{eps2pgf_color}\pgfsetfillcolor{eps2pgf_color}
\pgfpathmoveto{\pgfqpoint{0.117cm}{1.421cm}}
\pgfpathlineto{\pgfqpoint{0.682cm}{0.671cm}}
\pgfpathlineto{\pgfqpoint{1.246cm}{1.421cm}}
\pgfusepath{stroke}
\end{pgfscope}
\definecolor{eps2pgf_color}{gray}{0}\pgfsetstrokecolor{eps2pgf_color}\pgfsetfillcolor{eps2pgf_color}
\pgfpathmoveto{\pgfqpoint{0.273cm}{1.395cm}}
\pgfpathcurveto{\pgfqpoint{0.273cm}{1.432cm}}{\pgfqpoint{0.259cm}{1.467cm}}{\pgfqpoint{0.233cm}{1.492cm}}
\pgfpathcurveto{\pgfqpoint{0.207cm}{1.518cm}}{\pgfqpoint{0.173cm}{1.532cm}}{\pgfqpoint{0.137cm}{1.532cm}}
\pgfpathcurveto{\pgfqpoint{0.1cm}{1.532cm}}{\pgfqpoint{0.066cm}{1.518cm}}{\pgfqpoint{0.04cm}{1.492cm}}
\pgfpathcurveto{\pgfqpoint{0.014cm}{1.467cm}}{\pgfqpoint{0cm}{1.432cm}}{\pgfqpoint{0cm}{1.395cm}}
\pgfpathcurveto{\pgfqpoint{0cm}{1.359cm}}{\pgfqpoint{0.014cm}{1.324cm}}{\pgfqpoint{0.04cm}{1.299cm}}
\pgfpathcurveto{\pgfqpoint{0.066cm}{1.273cm}}{\pgfqpoint{0.1cm}{1.258cm}}{\pgfqpoint{0.137cm}{1.258cm}}
\pgfpathcurveto{\pgfqpoint{0.173cm}{1.258cm}}{\pgfqpoint{0.207cm}{1.273cm}}{\pgfqpoint{0.233cm}{1.299cm}}
\pgfpathcurveto{\pgfqpoint{0.259cm}{1.324cm}}{\pgfqpoint{0.273cm}{1.359cm}}{\pgfqpoint{0.273cm}{1.395cm}}
\pgfusepath{fill}
\begin{pgfscope}
\pgfsetdash{}{0cm}
\pgfsetlinewidth{0.818mm}
\pgfsetmiterlimit{7.0}
\pgfpathmoveto{\pgfqpoint{0.682cm}{0.671cm}}
\pgfpathlineto{\pgfqpoint{0.679cm}{1.418cm}}
\pgfusepath{stroke}
\end{pgfscope}
\pgfpathmoveto{\pgfqpoint{0.815cm}{1.399cm}}
\pgfpathcurveto{\pgfqpoint{0.815cm}{1.435cm}}{\pgfqpoint{0.801cm}{1.47cm}}{\pgfqpoint{0.775cm}{1.496cm}}
\pgfpathcurveto{\pgfqpoint{0.75cm}{1.521cm}}{\pgfqpoint{0.715cm}{1.536cm}}{\pgfqpoint{0.679cm}{1.536cm}}
\pgfpathcurveto{\pgfqpoint{0.643cm}{1.536cm}}{\pgfqpoint{0.608cm}{1.521cm}}{\pgfqpoint{0.582cm}{1.496cm}}
\pgfpathcurveto{\pgfqpoint{0.557cm}{1.47cm}}{\pgfqpoint{0.542cm}{1.435cm}}{\pgfqpoint{0.542cm}{1.399cm}}
\pgfpathcurveto{\pgfqpoint{0.542cm}{1.363cm}}{\pgfqpoint{0.557cm}{1.328cm}}{\pgfqpoint{0.582cm}{1.302cm}}
\pgfpathcurveto{\pgfqpoint{0.608cm}{1.276cm}}{\pgfqpoint{0.643cm}{1.262cm}}{\pgfqpoint{0.679cm}{1.262cm}}
\pgfpathcurveto{\pgfqpoint{0.715cm}{1.262cm}}{\pgfqpoint{0.75cm}{1.276cm}}{\pgfqpoint{0.775cm}{1.302cm}}
\pgfpathcurveto{\pgfqpoint{0.801cm}{1.328cm}}{\pgfqpoint{0.815cm}{1.363cm}}{\pgfqpoint{0.815cm}{1.399cm}}
\pgfusepath{fill}
\pgfpathmoveto{\pgfqpoint{1.345cm}{1.371cm}}
\pgfpathcurveto{\pgfqpoint{1.345cm}{1.408cm}}{\pgfqpoint{1.331cm}{1.442cm}}{\pgfqpoint{1.305cm}{1.468cm}}
\pgfpathcurveto{\pgfqpoint{1.28cm}{1.494cm}}{\pgfqpoint{1.245cm}{1.508cm}}{\pgfqpoint{1.209cm}{1.508cm}}
\pgfpathcurveto{\pgfqpoint{1.172cm}{1.508cm}}{\pgfqpoint{1.138cm}{1.494cm}}{\pgfqpoint{1.112cm}{1.468cm}}
\pgfpathcurveto{\pgfqpoint{1.087cm}{1.442cm}}{\pgfqpoint{1.072cm}{1.408cm}}{\pgfqpoint{1.072cm}{1.371cm}}
\pgfpathcurveto{\pgfqpoint{1.072cm}{1.335cm}}{\pgfqpoint{1.087cm}{1.3cm}}{\pgfqpoint{1.112cm}{1.274cm}}
\pgfpathcurveto{\pgfqpoint{1.138cm}{1.249cm}}{\pgfqpoint{1.172cm}{1.234cm}}{\pgfqpoint{1.209cm}{1.234cm}}
\pgfpathcurveto{\pgfqpoint{1.245cm}{1.234cm}}{\pgfqpoint{1.28cm}{1.249cm}}{\pgfqpoint{1.305cm}{1.274cm}}
\pgfpathcurveto{\pgfqpoint{1.331cm}{1.3cm}}{\pgfqpoint{1.345cm}{1.335cm}}{\pgfqpoint{1.345cm}{1.371cm}}
\pgfusepath{fill}
\begin{pgfscope}
\pgfsetdash{}{0cm}
\pgfsetlinewidth{0.818mm}
\pgfsetroundcap
\pgfsetmiterlimit{4.0}
\pgfpathmoveto{\pgfqpoint{0.682cm}{0.671cm}}
\pgfpathlineto{\pgfqpoint{0.682cm}{0.042cm}}
\pgfusepath{stroke}
\end{pgfscope}
\end{pgfscope}
\end{pgfscope}
\end{pgfscope}
\end{tikzpicture}}})$  a canonical representative of the random variables under consideration and let 
$
\zeta\assign\varphi-X+\lambdaX^{\!\resizebox{0.6em}{!}{
\begin{tikzpicture}
\pgfpathmoveto{\pgfqpoint{0cm}{-0.035cm}}
\pgfpathlineto{\pgfqpoint{1.376cm}{-0.035cm}}
\pgfpathlineto{\pgfqpoint{1.376cm}{1.552cm}}
\pgfpathlineto{\pgfqpoint{0cm}{1.552cm}}
\pgfpathclose
\pgfusepath{clip}
\begin{pgfscope}
\begin{pgfscope}
\pgfpathmoveto{\pgfqpoint{0cm}{-0.035cm}}
\pgfpathlineto{\pgfqpoint{1.376cm}{-0.035cm}}
\pgfpathlineto{\pgfqpoint{1.376cm}{1.552cm}}
\pgfpathlineto{\pgfqpoint{0cm}{1.552cm}}
\pgfpathclose
\pgfusepath{clip}
\begin{pgfscope}
\begin{pgfscope}
\pgfsetdash{}{0cm}
\pgfsetlinewidth{0.818mm}
\pgfsetroundcap
\pgfsetroundjoin
\pgfsetmiterlimit{7.0}
\definecolor{eps2pgf_color}{gray}{0}\pgfsetstrokecolor{eps2pgf_color}\pgfsetfillcolor{eps2pgf_color}
\pgfpathmoveto{\pgfqpoint{0.117cm}{1.421cm}}
\pgfpathlineto{\pgfqpoint{0.682cm}{0.671cm}}
\pgfpathlineto{\pgfqpoint{1.246cm}{1.421cm}}
\pgfusepath{stroke}
\end{pgfscope}
\definecolor{eps2pgf_color}{gray}{0}\pgfsetstrokecolor{eps2pgf_color}\pgfsetfillcolor{eps2pgf_color}
\pgfpathmoveto{\pgfqpoint{0.273cm}{1.395cm}}
\pgfpathcurveto{\pgfqpoint{0.273cm}{1.432cm}}{\pgfqpoint{0.259cm}{1.467cm}}{\pgfqpoint{0.233cm}{1.492cm}}
\pgfpathcurveto{\pgfqpoint{0.207cm}{1.518cm}}{\pgfqpoint{0.173cm}{1.532cm}}{\pgfqpoint{0.137cm}{1.532cm}}
\pgfpathcurveto{\pgfqpoint{0.1cm}{1.532cm}}{\pgfqpoint{0.066cm}{1.518cm}}{\pgfqpoint{0.04cm}{1.492cm}}
\pgfpathcurveto{\pgfqpoint{0.014cm}{1.467cm}}{\pgfqpoint{0cm}{1.432cm}}{\pgfqpoint{0cm}{1.395cm}}
\pgfpathcurveto{\pgfqpoint{0cm}{1.359cm}}{\pgfqpoint{0.014cm}{1.324cm}}{\pgfqpoint{0.04cm}{1.299cm}}
\pgfpathcurveto{\pgfqpoint{0.066cm}{1.273cm}}{\pgfqpoint{0.1cm}{1.258cm}}{\pgfqpoint{0.137cm}{1.258cm}}
\pgfpathcurveto{\pgfqpoint{0.173cm}{1.258cm}}{\pgfqpoint{0.207cm}{1.273cm}}{\pgfqpoint{0.233cm}{1.299cm}}
\pgfpathcurveto{\pgfqpoint{0.259cm}{1.324cm}}{\pgfqpoint{0.273cm}{1.359cm}}{\pgfqpoint{0.273cm}{1.395cm}}
\pgfusepath{fill}
\begin{pgfscope}
\pgfsetdash{}{0cm}
\pgfsetlinewidth{0.818mm}
\pgfsetmiterlimit{7.0}
\pgfpathmoveto{\pgfqpoint{0.682cm}{0.671cm}}
\pgfpathlineto{\pgfqpoint{0.679cm}{1.418cm}}
\pgfusepath{stroke}
\end{pgfscope}
\pgfpathmoveto{\pgfqpoint{0.815cm}{1.399cm}}
\pgfpathcurveto{\pgfqpoint{0.815cm}{1.435cm}}{\pgfqpoint{0.801cm}{1.47cm}}{\pgfqpoint{0.775cm}{1.496cm}}
\pgfpathcurveto{\pgfqpoint{0.75cm}{1.521cm}}{\pgfqpoint{0.715cm}{1.536cm}}{\pgfqpoint{0.679cm}{1.536cm}}
\pgfpathcurveto{\pgfqpoint{0.643cm}{1.536cm}}{\pgfqpoint{0.608cm}{1.521cm}}{\pgfqpoint{0.582cm}{1.496cm}}
\pgfpathcurveto{\pgfqpoint{0.557cm}{1.47cm}}{\pgfqpoint{0.542cm}{1.435cm}}{\pgfqpoint{0.542cm}{1.399cm}}
\pgfpathcurveto{\pgfqpoint{0.542cm}{1.363cm}}{\pgfqpoint{0.557cm}{1.328cm}}{\pgfqpoint{0.582cm}{1.302cm}}
\pgfpathcurveto{\pgfqpoint{0.608cm}{1.276cm}}{\pgfqpoint{0.643cm}{1.262cm}}{\pgfqpoint{0.679cm}{1.262cm}}
\pgfpathcurveto{\pgfqpoint{0.715cm}{1.262cm}}{\pgfqpoint{0.75cm}{1.276cm}}{\pgfqpoint{0.775cm}{1.302cm}}
\pgfpathcurveto{\pgfqpoint{0.801cm}{1.328cm}}{\pgfqpoint{0.815cm}{1.363cm}}{\pgfqpoint{0.815cm}{1.399cm}}
\pgfusepath{fill}
\pgfpathmoveto{\pgfqpoint{1.345cm}{1.371cm}}
\pgfpathcurveto{\pgfqpoint{1.345cm}{1.408cm}}{\pgfqpoint{1.331cm}{1.442cm}}{\pgfqpoint{1.305cm}{1.468cm}}
\pgfpathcurveto{\pgfqpoint{1.28cm}{1.494cm}}{\pgfqpoint{1.245cm}{1.508cm}}{\pgfqpoint{1.209cm}{1.508cm}}
\pgfpathcurveto{\pgfqpoint{1.172cm}{1.508cm}}{\pgfqpoint{1.138cm}{1.494cm}}{\pgfqpoint{1.112cm}{1.468cm}}
\pgfpathcurveto{\pgfqpoint{1.087cm}{1.442cm}}{\pgfqpoint{1.072cm}{1.408cm}}{\pgfqpoint{1.072cm}{1.371cm}}
\pgfpathcurveto{\pgfqpoint{1.072cm}{1.335cm}}{\pgfqpoint{1.087cm}{1.3cm}}{\pgfqpoint{1.112cm}{1.274cm}}
\pgfpathcurveto{\pgfqpoint{1.138cm}{1.249cm}}{\pgfqpoint{1.172cm}{1.234cm}}{\pgfqpoint{1.209cm}{1.234cm}}
\pgfpathcurveto{\pgfqpoint{1.245cm}{1.234cm}}{\pgfqpoint{1.28cm}{1.249cm}}{\pgfqpoint{1.305cm}{1.274cm}}
\pgfpathcurveto{\pgfqpoint{1.331cm}{1.3cm}}{\pgfqpoint{1.345cm}{1.335cm}}{\pgfqpoint{1.345cm}{1.371cm}}
\pgfusepath{fill}
\begin{pgfscope}
\pgfsetdash{}{0cm}
\pgfsetlinewidth{0.818mm}
\pgfsetroundcap
\pgfsetmiterlimit{4.0}
\pgfpathmoveto{\pgfqpoint{0.682cm}{0.671cm}}
\pgfpathlineto{\pgfqpoint{0.682cm}{0.042cm}}
\pgfusepath{stroke}
\end{pgfscope}
\end{pgfscope}
\end{pgfscope}
\end{pgfscope}
\end{tikzpicture}}}.
$

\begin{theorem}\label{thm:main}
Let $\rho$ be a weight such that $\rho^{\iota} \in L^4$ for
  some $\iota \in (0, 1)$. Then 
  the family of joint laws of $( \mathcal{E}^{\varepsilon} \varphi_{M, \varepsilon}(t),
\mathcal{E}^{\varepsilon} X_{M, \varepsilon}(t),\mathcal{E}^{\varepsilon} X^{\!\resizebox{0.6em}{!}{
\begin{tikzpicture}
\pgfpathmoveto{\pgfqpoint{0cm}{-0.035cm}}
\pgfpathlineto{\pgfqpoint{1.376cm}{-0.035cm}}
\pgfpathlineto{\pgfqpoint{1.376cm}{1.552cm}}
\pgfpathlineto{\pgfqpoint{0cm}{1.552cm}}
\pgfpathclose
\pgfusepath{clip}
\begin{pgfscope}
\begin{pgfscope}
\pgfpathmoveto{\pgfqpoint{0cm}{-0.035cm}}
\pgfpathlineto{\pgfqpoint{1.376cm}{-0.035cm}}
\pgfpathlineto{\pgfqpoint{1.376cm}{1.552cm}}
\pgfpathlineto{\pgfqpoint{0cm}{1.552cm}}
\pgfpathclose
\pgfusepath{clip}
\begin{pgfscope}
\begin{pgfscope}
\pgfsetdash{}{0cm}
\pgfsetlinewidth{0.818mm}
\pgfsetroundcap
\pgfsetroundjoin
\pgfsetmiterlimit{7.0}
\definecolor{eps2pgf_color}{gray}{0}\pgfsetstrokecolor{eps2pgf_color}\pgfsetfillcolor{eps2pgf_color}
\pgfpathmoveto{\pgfqpoint{0.117cm}{1.421cm}}
\pgfpathlineto{\pgfqpoint{0.682cm}{0.671cm}}
\pgfpathlineto{\pgfqpoint{1.246cm}{1.421cm}}
\pgfusepath{stroke}
\end{pgfscope}
\definecolor{eps2pgf_color}{gray}{0}\pgfsetstrokecolor{eps2pgf_color}\pgfsetfillcolor{eps2pgf_color}
\pgfpathmoveto{\pgfqpoint{0.273cm}{1.395cm}}
\pgfpathcurveto{\pgfqpoint{0.273cm}{1.432cm}}{\pgfqpoint{0.259cm}{1.467cm}}{\pgfqpoint{0.233cm}{1.492cm}}
\pgfpathcurveto{\pgfqpoint{0.207cm}{1.518cm}}{\pgfqpoint{0.173cm}{1.532cm}}{\pgfqpoint{0.137cm}{1.532cm}}
\pgfpathcurveto{\pgfqpoint{0.1cm}{1.532cm}}{\pgfqpoint{0.066cm}{1.518cm}}{\pgfqpoint{0.04cm}{1.492cm}}
\pgfpathcurveto{\pgfqpoint{0.014cm}{1.467cm}}{\pgfqpoint{0cm}{1.432cm}}{\pgfqpoint{0cm}{1.395cm}}
\pgfpathcurveto{\pgfqpoint{0cm}{1.359cm}}{\pgfqpoint{0.014cm}{1.324cm}}{\pgfqpoint{0.04cm}{1.299cm}}
\pgfpathcurveto{\pgfqpoint{0.066cm}{1.273cm}}{\pgfqpoint{0.1cm}{1.258cm}}{\pgfqpoint{0.137cm}{1.258cm}}
\pgfpathcurveto{\pgfqpoint{0.173cm}{1.258cm}}{\pgfqpoint{0.207cm}{1.273cm}}{\pgfqpoint{0.233cm}{1.299cm}}
\pgfpathcurveto{\pgfqpoint{0.259cm}{1.324cm}}{\pgfqpoint{0.273cm}{1.359cm}}{\pgfqpoint{0.273cm}{1.395cm}}
\pgfusepath{fill}
\begin{pgfscope}
\pgfsetdash{}{0cm}
\pgfsetlinewidth{0.818mm}
\pgfsetmiterlimit{7.0}
\pgfpathmoveto{\pgfqpoint{0.682cm}{0.671cm}}
\pgfpathlineto{\pgfqpoint{0.679cm}{1.418cm}}
\pgfusepath{stroke}
\end{pgfscope}
\pgfpathmoveto{\pgfqpoint{0.815cm}{1.399cm}}
\pgfpathcurveto{\pgfqpoint{0.815cm}{1.435cm}}{\pgfqpoint{0.801cm}{1.47cm}}{\pgfqpoint{0.775cm}{1.496cm}}
\pgfpathcurveto{\pgfqpoint{0.75cm}{1.521cm}}{\pgfqpoint{0.715cm}{1.536cm}}{\pgfqpoint{0.679cm}{1.536cm}}
\pgfpathcurveto{\pgfqpoint{0.643cm}{1.536cm}}{\pgfqpoint{0.608cm}{1.521cm}}{\pgfqpoint{0.582cm}{1.496cm}}
\pgfpathcurveto{\pgfqpoint{0.557cm}{1.47cm}}{\pgfqpoint{0.542cm}{1.435cm}}{\pgfqpoint{0.542cm}{1.399cm}}
\pgfpathcurveto{\pgfqpoint{0.542cm}{1.363cm}}{\pgfqpoint{0.557cm}{1.328cm}}{\pgfqpoint{0.582cm}{1.302cm}}
\pgfpathcurveto{\pgfqpoint{0.608cm}{1.276cm}}{\pgfqpoint{0.643cm}{1.262cm}}{\pgfqpoint{0.679cm}{1.262cm}}
\pgfpathcurveto{\pgfqpoint{0.715cm}{1.262cm}}{\pgfqpoint{0.75cm}{1.276cm}}{\pgfqpoint{0.775cm}{1.302cm}}
\pgfpathcurveto{\pgfqpoint{0.801cm}{1.328cm}}{\pgfqpoint{0.815cm}{1.363cm}}{\pgfqpoint{0.815cm}{1.399cm}}
\pgfusepath{fill}
\pgfpathmoveto{\pgfqpoint{1.345cm}{1.371cm}}
\pgfpathcurveto{\pgfqpoint{1.345cm}{1.408cm}}{\pgfqpoint{1.331cm}{1.442cm}}{\pgfqpoint{1.305cm}{1.468cm}}
\pgfpathcurveto{\pgfqpoint{1.28cm}{1.494cm}}{\pgfqpoint{1.245cm}{1.508cm}}{\pgfqpoint{1.209cm}{1.508cm}}
\pgfpathcurveto{\pgfqpoint{1.172cm}{1.508cm}}{\pgfqpoint{1.138cm}{1.494cm}}{\pgfqpoint{1.112cm}{1.468cm}}
\pgfpathcurveto{\pgfqpoint{1.087cm}{1.442cm}}{\pgfqpoint{1.072cm}{1.408cm}}{\pgfqpoint{1.072cm}{1.371cm}}
\pgfpathcurveto{\pgfqpoint{1.072cm}{1.335cm}}{\pgfqpoint{1.087cm}{1.3cm}}{\pgfqpoint{1.112cm}{1.274cm}}
\pgfpathcurveto{\pgfqpoint{1.138cm}{1.249cm}}{\pgfqpoint{1.172cm}{1.234cm}}{\pgfqpoint{1.209cm}{1.234cm}}
\pgfpathcurveto{\pgfqpoint{1.245cm}{1.234cm}}{\pgfqpoint{1.28cm}{1.249cm}}{\pgfqpoint{1.305cm}{1.274cm}}
\pgfpathcurveto{\pgfqpoint{1.331cm}{1.3cm}}{\pgfqpoint{1.345cm}{1.335cm}}{\pgfqpoint{1.345cm}{1.371cm}}
\pgfusepath{fill}
\begin{pgfscope}
\pgfsetdash{}{0cm}
\pgfsetlinewidth{0.818mm}
\pgfsetroundcap
\pgfsetmiterlimit{4.0}
\pgfpathmoveto{\pgfqpoint{0.682cm}{0.671cm}}
\pgfpathlineto{\pgfqpoint{0.682cm}{0.042cm}}
\pgfusepath{stroke}
\end{pgfscope}
\end{pgfscope}
\end{pgfscope}
\end{pgfscope}
\end{tikzpicture}}}_{M, \varepsilon}(t))$, $\varepsilon\in\mathcal{A},M>0,$ evaluated at an arbitrary time $t \geqslant 0$ is tight on $H^{-1/2-3\kappa}(\rho^{2+\kappa})\times \CC^{-1/2-\kappa}(\rho^{\sigma})\times \CC^{1/2-\kappa}(\rho^{\sigma})$. Moreover, any limit probability measure $\mu$    satisfies for all $p \in [1, \infty)$
  \[ \mathbb{E}_{\mu} \| \varphi \|_{H^{- 1 / 2 - 2\kappa} (\rho^2)}^{2 p} \lesssim  {1+\lambda^{3p}}, \qquad \mathbb{E}_{\mu} \| \zeta  \|^{2p}_{L^{2}
   (\rho^2)} \lesssim {\lambda^{p}+\lambda^{3p+4}+\lambda^{4p}}, \]
\[\mathbb{E}_{\mu} \| \zeta \|_{H^{1 - 2 \kappa}
   (\rho^2)}^2  \lesssim \lambda^2 +\lambda^{7},\qquad
   \mathbb{E}_{\mu} \| \zeta  \|_{B^{0}_{4, \infty}
   (\rho)}^4  \lesssim {\lambda +\lambda^{6}}.
\]

\end{theorem}

\begin{proof}
 Since by Lemma~\ref{lem:ext}
  \[ \mathbb{E} \|
     \mathcal{E}^{\varepsilon} X_{M, \varepsilon} (0) \|_{H^{- 1 / 2 - 2
     \kappa} (\rho^2)}^{2p} \lesssim 
     \mathbb{E} \| X_{M, \varepsilon} (0) \|_{\CC^{- 1 / 2 - \kappa,
     \varepsilon} (\rho^{\sigma})}^{2p} \lesssim 1, \]
     uniformly in $M,\varepsilon$,  we deduce from Theorem~\ref{thm:tight} that
   \[ 
   \mathbb{E} \|\mathcal{E}^{\varepsilon}
     \varphi_{M, \varepsilon} (0) \|_{H^{- 1 / 2 - 2
     \kappa} (\rho^2)}^{2 p} \lesssim  {1+\lambda^{3p}}  \]
 uniformly in $M,\varepsilon$.
  Integrating \eqref{eq:lp} in time and using the decomposition of $\varphi_{M,\varepsilon}$ leads to
\[ \| \rho^2 \phi_{M,\varepsilon} (t) \|_{L^{2, \varepsilon}}^{2 p} \leqslant \|
   \rho^2 \phi_{M,\varepsilon} (0) \|_{L^{2, \varepsilon}}^{2 p} + C_t \lambda (\lambda^{2}+\lambda^{6})^{(p+1)/2} Q_{\rho}
    (\mathbb{X}_{M,\varepsilon})^{(p + 1) / 2} \]
\[ \leqslant C_p \| \rho^2 (\varphi_{M,\varepsilon} (0) - X_{M,\varepsilon} (0))
    \|_{L^{2, \varepsilon}}^{2 p} + C_p \| \rho^2 Y_{M,\varepsilon} (0) \|_{L^{2,
    \varepsilon}}^{2 p} + C_t \lambda (\lambda^{2}+\lambda^{6})^{(p+1)/2} Q_{\rho} (\mathbb{X}_{M,\varepsilon})^{(p + 1) / 2}.
\]
Hence due to Theorem~\ref{thm:tight} we obtain a uniform bound 
\[
 \mathbb{E} \| \rho^2 \phi_{M,\varepsilon} (t) \|_{L^{2, \varepsilon}}^{2 p} \lesssim_t {\lambda^{p}+\lambda^{3p+4}},\]
for all $t\geqslant 0$.
In addition, the following expressions are bounded uniformly in
$M, \varepsilon$ according to Lemma~\ref{lem:Y1} and Theorem~\ref{th:energy-estimate}
\[ \mathbb{E} \| \eta_{M, \varepsilon} \|_{C_T \CC^{1 - \kappa, \varepsilon}
   (\rho^{\sigma})}^{2p} \lesssim \lambda^{{4p}},
\]
\[ \lambda \int_0^T \mathbb{E} \| \phi_{M, \varepsilon} (t)
   \|_{L^{4, \varepsilon} (\rho)}^4 \mathd t + \int_0^T \mathbb{E} \|
   \phi_{M, \varepsilon} (t) \|_{H^{1 - 2 \kappa, \varepsilon} (\rho^2)}^2
   \mathd t \lesssim_T \lambda^2 + \lambda^{7},\]
whenever the weight $\rho$ is such
that $\rho^{\iota} \in L^4$ for some $\iota \in (0, 1)$. In view of
stationarity of $\zeta_{M, \varepsilon}$ and the embedding $\CC^{1 -
\kappa, \varepsilon} (\rho^{\sigma}) \subset H^{1 - 2 \kappa, \varepsilon}
(\rho^2)$, we therefore obtain a uniform bound
$ \mathbb{E} \| \zeta_{M, \varepsilon} (t) \|_{H^{1 - 2 \kappa, \varepsilon}
   (\rho^2)}^2 \lesssim \lambda^2 + \lambda^{7}$ as well as $ \mathbb{E} \| \zeta_{M, \varepsilon} (t) \|_{L^{2,\varepsilon}
   (\rho^2)}^{2p} \lesssim {\lambda^{p}+\lambda^{3p+4}+\lambda^{4p}}$
for every $t \geqslant 0$. Similarly, using stationarity together with the
embedding $\CC^{1 - \kappa, \varepsilon} (\rho^{\sigma}) \subset B^{0,
\varepsilon}_{4, \infty} (\rho)$ as well as $L^{4, \varepsilon} (\rho) \subset
B^{0, \varepsilon}_{4, \infty}(\rho)$ we deduce a uniform bound
$\mathbb{E} \| \zeta_{M, \varepsilon} (t) \|_{B^{0, \varepsilon}_{4, \infty}
   (\rho)}^4  \lesssim {\lambda + \lambda^{6}}$
for every $t \geqslant 0$.

Consequently, by Lemma~\ref{lem:ext} the same
bounds hold for the corresponding extended distributions and hence the family joint laws of
$( \mathcal{E}^{\varepsilon} \varphi_{M, \varepsilon}(t),
\mathcal{E}^{\varepsilon} X_{M, \varepsilon}(t),\mathcal{E}^{\varepsilon} X^{\!\resizebox{0.6em}{!}{
\begin{tikzpicture}
\pgfpathmoveto{\pgfqpoint{0cm}{-0.035cm}}
\pgfpathlineto{\pgfqpoint{1.376cm}{-0.035cm}}
\pgfpathlineto{\pgfqpoint{1.376cm}{1.552cm}}
\pgfpathlineto{\pgfqpoint{0cm}{1.552cm}}
\pgfpathclose
\pgfusepath{clip}
\begin{pgfscope}
\begin{pgfscope}
\pgfpathmoveto{\pgfqpoint{0cm}{-0.035cm}}
\pgfpathlineto{\pgfqpoint{1.376cm}{-0.035cm}}
\pgfpathlineto{\pgfqpoint{1.376cm}{1.552cm}}
\pgfpathlineto{\pgfqpoint{0cm}{1.552cm}}
\pgfpathclose
\pgfusepath{clip}
\begin{pgfscope}
\begin{pgfscope}
\pgfsetdash{}{0cm}
\pgfsetlinewidth{0.818mm}
\pgfsetroundcap
\pgfsetroundjoin
\pgfsetmiterlimit{7.0}
\definecolor{eps2pgf_color}{gray}{0}\pgfsetstrokecolor{eps2pgf_color}\pgfsetfillcolor{eps2pgf_color}
\pgfpathmoveto{\pgfqpoint{0.117cm}{1.421cm}}
\pgfpathlineto{\pgfqpoint{0.682cm}{0.671cm}}
\pgfpathlineto{\pgfqpoint{1.246cm}{1.421cm}}
\pgfusepath{stroke}
\end{pgfscope}
\definecolor{eps2pgf_color}{gray}{0}\pgfsetstrokecolor{eps2pgf_color}\pgfsetfillcolor{eps2pgf_color}
\pgfpathmoveto{\pgfqpoint{0.273cm}{1.395cm}}
\pgfpathcurveto{\pgfqpoint{0.273cm}{1.432cm}}{\pgfqpoint{0.259cm}{1.467cm}}{\pgfqpoint{0.233cm}{1.492cm}}
\pgfpathcurveto{\pgfqpoint{0.207cm}{1.518cm}}{\pgfqpoint{0.173cm}{1.532cm}}{\pgfqpoint{0.137cm}{1.532cm}}
\pgfpathcurveto{\pgfqpoint{0.1cm}{1.532cm}}{\pgfqpoint{0.066cm}{1.518cm}}{\pgfqpoint{0.04cm}{1.492cm}}
\pgfpathcurveto{\pgfqpoint{0.014cm}{1.467cm}}{\pgfqpoint{0cm}{1.432cm}}{\pgfqpoint{0cm}{1.395cm}}
\pgfpathcurveto{\pgfqpoint{0cm}{1.359cm}}{\pgfqpoint{0.014cm}{1.324cm}}{\pgfqpoint{0.04cm}{1.299cm}}
\pgfpathcurveto{\pgfqpoint{0.066cm}{1.273cm}}{\pgfqpoint{0.1cm}{1.258cm}}{\pgfqpoint{0.137cm}{1.258cm}}
\pgfpathcurveto{\pgfqpoint{0.173cm}{1.258cm}}{\pgfqpoint{0.207cm}{1.273cm}}{\pgfqpoint{0.233cm}{1.299cm}}
\pgfpathcurveto{\pgfqpoint{0.259cm}{1.324cm}}{\pgfqpoint{0.273cm}{1.359cm}}{\pgfqpoint{0.273cm}{1.395cm}}
\pgfusepath{fill}
\begin{pgfscope}
\pgfsetdash{}{0cm}
\pgfsetlinewidth{0.818mm}
\pgfsetmiterlimit{7.0}
\pgfpathmoveto{\pgfqpoint{0.682cm}{0.671cm}}
\pgfpathlineto{\pgfqpoint{0.679cm}{1.418cm}}
\pgfusepath{stroke}
\end{pgfscope}
\pgfpathmoveto{\pgfqpoint{0.815cm}{1.399cm}}
\pgfpathcurveto{\pgfqpoint{0.815cm}{1.435cm}}{\pgfqpoint{0.801cm}{1.47cm}}{\pgfqpoint{0.775cm}{1.496cm}}
\pgfpathcurveto{\pgfqpoint{0.75cm}{1.521cm}}{\pgfqpoint{0.715cm}{1.536cm}}{\pgfqpoint{0.679cm}{1.536cm}}
\pgfpathcurveto{\pgfqpoint{0.643cm}{1.536cm}}{\pgfqpoint{0.608cm}{1.521cm}}{\pgfqpoint{0.582cm}{1.496cm}}
\pgfpathcurveto{\pgfqpoint{0.557cm}{1.47cm}}{\pgfqpoint{0.542cm}{1.435cm}}{\pgfqpoint{0.542cm}{1.399cm}}
\pgfpathcurveto{\pgfqpoint{0.542cm}{1.363cm}}{\pgfqpoint{0.557cm}{1.328cm}}{\pgfqpoint{0.582cm}{1.302cm}}
\pgfpathcurveto{\pgfqpoint{0.608cm}{1.276cm}}{\pgfqpoint{0.643cm}{1.262cm}}{\pgfqpoint{0.679cm}{1.262cm}}
\pgfpathcurveto{\pgfqpoint{0.715cm}{1.262cm}}{\pgfqpoint{0.75cm}{1.276cm}}{\pgfqpoint{0.775cm}{1.302cm}}
\pgfpathcurveto{\pgfqpoint{0.801cm}{1.328cm}}{\pgfqpoint{0.815cm}{1.363cm}}{\pgfqpoint{0.815cm}{1.399cm}}
\pgfusepath{fill}
\pgfpathmoveto{\pgfqpoint{1.345cm}{1.371cm}}
\pgfpathcurveto{\pgfqpoint{1.345cm}{1.408cm}}{\pgfqpoint{1.331cm}{1.442cm}}{\pgfqpoint{1.305cm}{1.468cm}}
\pgfpathcurveto{\pgfqpoint{1.28cm}{1.494cm}}{\pgfqpoint{1.245cm}{1.508cm}}{\pgfqpoint{1.209cm}{1.508cm}}
\pgfpathcurveto{\pgfqpoint{1.172cm}{1.508cm}}{\pgfqpoint{1.138cm}{1.494cm}}{\pgfqpoint{1.112cm}{1.468cm}}
\pgfpathcurveto{\pgfqpoint{1.087cm}{1.442cm}}{\pgfqpoint{1.072cm}{1.408cm}}{\pgfqpoint{1.072cm}{1.371cm}}
\pgfpathcurveto{\pgfqpoint{1.072cm}{1.335cm}}{\pgfqpoint{1.087cm}{1.3cm}}{\pgfqpoint{1.112cm}{1.274cm}}
\pgfpathcurveto{\pgfqpoint{1.138cm}{1.249cm}}{\pgfqpoint{1.172cm}{1.234cm}}{\pgfqpoint{1.209cm}{1.234cm}}
\pgfpathcurveto{\pgfqpoint{1.245cm}{1.234cm}}{\pgfqpoint{1.28cm}{1.249cm}}{\pgfqpoint{1.305cm}{1.274cm}}
\pgfpathcurveto{\pgfqpoint{1.331cm}{1.3cm}}{\pgfqpoint{1.345cm}{1.335cm}}{\pgfqpoint{1.345cm}{1.371cm}}
\pgfusepath{fill}
\begin{pgfscope}
\pgfsetdash{}{0cm}
\pgfsetlinewidth{0.818mm}
\pgfsetroundcap
\pgfsetmiterlimit{4.0}
\pgfpathmoveto{\pgfqpoint{0.682cm}{0.671cm}}
\pgfpathlineto{\pgfqpoint{0.682cm}{0.042cm}}
\pgfusepath{stroke}
\end{pgfscope}
\end{pgfscope}
\end{pgfscope}
\end{pgfscope}
\end{tikzpicture}}}_{M, \varepsilon}(t))$ at any time $t\geqslant 0$ is tight on $H^{-1/2-3\kappa}(\rho^{2+\kappa})\times \CC^{-1/2-\kappa}(\rho^{\sigma})\times \CC^{1/2-\kappa}(\rho^{\sigma})$. Indeed, this is a consequence of the compact embedding
$$
H^{-1/2-2\kappa}(\rho^{2})\times \CC^{-1/2-\kappa/2}(\rho^{2\sigma})\times \CC^{1/2-\kappa/2}(\rho^{2\sigma})\subset H^{-1/2-3\kappa}(\rho^{2+\kappa})\times \CC^{-1/2-\kappa}(\rho^{\sigma})\times \CC^{1/2-\kappa}(\rho^{\sigma}).
$$
Therefore up to a subsequence we may pass to the
limit as $\varepsilon \rightarrow 0$, $M \rightarrow \infty$ and the uniform moment bounds are preserved for every limit point.
\end{proof}

The marginal of $\mu$ corresponding to $\varphi$ is the desired  $\Phi^{4}_3$ measure, which we denote by $\nu$. According to the above result, $\nu$ is obtained as a limit (up to a subsequence) of the continuum extensions of the Gibbs measures $\nu_{M,\varepsilon}$ given by \eqref{eq:gibbs} as $\varepsilon \rightarrow 0$, $M \rightarrow \infty$.

\subsection{Stretched exponential integrability}
\label{s:exp}

The goal of this section is to establish better probabilistic properties of the $\Phi^{4}_{3}$ measure. Namely, we show that  $\| \rho^2 \varphi_{M,
\varepsilon} \|_{H^{- 1 / 2 - 2\kappa, \varepsilon}}^{1 - \upsilon}$ is uniformly (in $M,\varepsilon$)
exponentially integrable for every  $\upsilon=O(\kappa) > 0$, hence we recover the same stretched exponential moment bound for any limit measure $\nu$.
To this end, we revisit the energy estimate in Section \ref{s:estim} and take a particular care  to optimize the power of the quantity $\|\mathbb{X}_{M,\varepsilon}\|$ appearing in the estimates. 
Recall that it can be shown that
\begin{equation}
  \mathbb{E} [e^{\beta \| \mathbb{X}_{M, \varepsilon} \|^2}] < \infty
  \label{eq:exp-int}
\end{equation}
uniformly in $M, \varepsilon$ for a small parameter $\beta > 0$ (see
{\cite{moinat_space_time_2018}}). Accordingly, it turns out that the polynomial $Q_{\rho}(\mathbb X_{M,\varepsilon})$ on the right hand side of the bound in Lemma \ref{lemma:bounds-rhs1} shall not contain higher powers of $\|\mathbb X_{M,\varepsilon}\|$ than $8+O(\kappa)$. In the proof of Lemma \ref{lemma:bounds-rhs1} we already see what the problematic terms are. In order to allow for a refined treatment of these terms, we introduce an additional large momentum cut-off and modify the definition of $Y_{M,\varepsilon}$ from \eqref{eq:Y1}, leading to better uniform estimates and consequently to the desired stretched exponential integrability.

More precisely, let $K
> 0$ and take a compactly supported, smooth function $v : \mathbb{R} \rightarrow
\mathbb{R}_+$ such that $\| v \|_{L^1} = 1$. We define
\[ \llbracket X_{M, \varepsilon}^3 \rrbracket_{\leqslant} \assign v_K \ast_t
   \Delta^{\varepsilon}_{\leqslant K} \llbracket X_{M, \varepsilon}^3
   \rrbracket, \]
where the convolution is in the time variable and $v_K (t) \assign 2^K v (2^K
t)$. With standard arguments one can prove that
\[ \sup_{K \in \mathbb{N}} (2^{- K (3 / 2 + \kappa)} \| \llbracket X_{M,
   \varepsilon}^3 \rrbracket_{\leqslant} \|_{C_T L^{\infty, \varepsilon}})^{2
   / 3} \]
is exponentially integrable for a small parameter and therefore we can modify the definition of $\| \mathbb{X}_{M, \varepsilon} \|$ to obtain
\begin{equation}
  \| \llbracket X_{M, \varepsilon}^3 \rrbracket_{\leqslant} \|_{C_T L^{\infty,
  \varepsilon}} \lesssim 2^{K (3 / 2 + \kappa)} \| \mathbb{X}_{M, \varepsilon}
  \|^3  \label{eq:bound-X3}
\end{equation}
while still keeping the validity of~{\eqref{eq:exp-int}}. Moreover, we let $\llbracket X_{M,
\varepsilon}^3 \rrbracket_{>} \assign \llbracket X_{M, \varepsilon}^3
\rrbracket - \llbracket X_{M, \varepsilon}^3 \rrbracket_{\leqslant}$ and
define $X_{M, \varepsilon, >}^{\!\resizebox{0.6em}{!}{
\begin{tikzpicture}
\pgfpathmoveto{\pgfqpoint{0cm}{-0.035cm}}
\pgfpathlineto{\pgfqpoint{1.376cm}{-0.035cm}}
\pgfpathlineto{\pgfqpoint{1.376cm}{1.552cm}}
\pgfpathlineto{\pgfqpoint{0cm}{1.552cm}}
\pgfpathclose
\pgfusepath{clip}
\begin{pgfscope}
\begin{pgfscope}
\pgfpathmoveto{\pgfqpoint{0cm}{-0.035cm}}
\pgfpathlineto{\pgfqpoint{1.376cm}{-0.035cm}}
\pgfpathlineto{\pgfqpoint{1.376cm}{1.552cm}}
\pgfpathlineto{\pgfqpoint{0cm}{1.552cm}}
\pgfpathclose
\pgfusepath{clip}
\begin{pgfscope}
\begin{pgfscope}
\pgfsetdash{}{0cm}
\pgfsetlinewidth{0.818mm}
\pgfsetroundcap
\pgfsetroundjoin
\pgfsetmiterlimit{7.0}
\definecolor{eps2pgf_color}{gray}{0}\pgfsetstrokecolor{eps2pgf_color}\pgfsetfillcolor{eps2pgf_color}
\pgfpathmoveto{\pgfqpoint{0.117cm}{1.421cm}}
\pgfpathlineto{\pgfqpoint{0.682cm}{0.671cm}}
\pgfpathlineto{\pgfqpoint{1.246cm}{1.421cm}}
\pgfusepath{stroke}
\end{pgfscope}
\definecolor{eps2pgf_color}{gray}{0}\pgfsetstrokecolor{eps2pgf_color}\pgfsetfillcolor{eps2pgf_color}
\pgfpathmoveto{\pgfqpoint{0.273cm}{1.395cm}}
\pgfpathcurveto{\pgfqpoint{0.273cm}{1.432cm}}{\pgfqpoint{0.259cm}{1.467cm}}{\pgfqpoint{0.233cm}{1.492cm}}
\pgfpathcurveto{\pgfqpoint{0.207cm}{1.518cm}}{\pgfqpoint{0.173cm}{1.532cm}}{\pgfqpoint{0.137cm}{1.532cm}}
\pgfpathcurveto{\pgfqpoint{0.1cm}{1.532cm}}{\pgfqpoint{0.066cm}{1.518cm}}{\pgfqpoint{0.04cm}{1.492cm}}
\pgfpathcurveto{\pgfqpoint{0.014cm}{1.467cm}}{\pgfqpoint{0cm}{1.432cm}}{\pgfqpoint{0cm}{1.395cm}}
\pgfpathcurveto{\pgfqpoint{0cm}{1.359cm}}{\pgfqpoint{0.014cm}{1.324cm}}{\pgfqpoint{0.04cm}{1.299cm}}
\pgfpathcurveto{\pgfqpoint{0.066cm}{1.273cm}}{\pgfqpoint{0.1cm}{1.258cm}}{\pgfqpoint{0.137cm}{1.258cm}}
\pgfpathcurveto{\pgfqpoint{0.173cm}{1.258cm}}{\pgfqpoint{0.207cm}{1.273cm}}{\pgfqpoint{0.233cm}{1.299cm}}
\pgfpathcurveto{\pgfqpoint{0.259cm}{1.324cm}}{\pgfqpoint{0.273cm}{1.359cm}}{\pgfqpoint{0.273cm}{1.395cm}}
\pgfusepath{fill}
\begin{pgfscope}
\pgfsetdash{}{0cm}
\pgfsetlinewidth{0.818mm}
\pgfsetmiterlimit{7.0}
\pgfpathmoveto{\pgfqpoint{0.682cm}{0.671cm}}
\pgfpathlineto{\pgfqpoint{0.679cm}{1.418cm}}
\pgfusepath{stroke}
\end{pgfscope}
\pgfpathmoveto{\pgfqpoint{0.815cm}{1.399cm}}
\pgfpathcurveto{\pgfqpoint{0.815cm}{1.435cm}}{\pgfqpoint{0.801cm}{1.47cm}}{\pgfqpoint{0.775cm}{1.496cm}}
\pgfpathcurveto{\pgfqpoint{0.75cm}{1.521cm}}{\pgfqpoint{0.715cm}{1.536cm}}{\pgfqpoint{0.679cm}{1.536cm}}
\pgfpathcurveto{\pgfqpoint{0.643cm}{1.536cm}}{\pgfqpoint{0.608cm}{1.521cm}}{\pgfqpoint{0.582cm}{1.496cm}}
\pgfpathcurveto{\pgfqpoint{0.557cm}{1.47cm}}{\pgfqpoint{0.542cm}{1.435cm}}{\pgfqpoint{0.542cm}{1.399cm}}
\pgfpathcurveto{\pgfqpoint{0.542cm}{1.363cm}}{\pgfqpoint{0.557cm}{1.328cm}}{\pgfqpoint{0.582cm}{1.302cm}}
\pgfpathcurveto{\pgfqpoint{0.608cm}{1.276cm}}{\pgfqpoint{0.643cm}{1.262cm}}{\pgfqpoint{0.679cm}{1.262cm}}
\pgfpathcurveto{\pgfqpoint{0.715cm}{1.262cm}}{\pgfqpoint{0.75cm}{1.276cm}}{\pgfqpoint{0.775cm}{1.302cm}}
\pgfpathcurveto{\pgfqpoint{0.801cm}{1.328cm}}{\pgfqpoint{0.815cm}{1.363cm}}{\pgfqpoint{0.815cm}{1.399cm}}
\pgfusepath{fill}
\pgfpathmoveto{\pgfqpoint{1.345cm}{1.371cm}}
\pgfpathcurveto{\pgfqpoint{1.345cm}{1.408cm}}{\pgfqpoint{1.331cm}{1.442cm}}{\pgfqpoint{1.305cm}{1.468cm}}
\pgfpathcurveto{\pgfqpoint{1.28cm}{1.494cm}}{\pgfqpoint{1.245cm}{1.508cm}}{\pgfqpoint{1.209cm}{1.508cm}}
\pgfpathcurveto{\pgfqpoint{1.172cm}{1.508cm}}{\pgfqpoint{1.138cm}{1.494cm}}{\pgfqpoint{1.112cm}{1.468cm}}
\pgfpathcurveto{\pgfqpoint{1.087cm}{1.442cm}}{\pgfqpoint{1.072cm}{1.408cm}}{\pgfqpoint{1.072cm}{1.371cm}}
\pgfpathcurveto{\pgfqpoint{1.072cm}{1.335cm}}{\pgfqpoint{1.087cm}{1.3cm}}{\pgfqpoint{1.112cm}{1.274cm}}
\pgfpathcurveto{\pgfqpoint{1.138cm}{1.249cm}}{\pgfqpoint{1.172cm}{1.234cm}}{\pgfqpoint{1.209cm}{1.234cm}}
\pgfpathcurveto{\pgfqpoint{1.245cm}{1.234cm}}{\pgfqpoint{1.28cm}{1.249cm}}{\pgfqpoint{1.305cm}{1.274cm}}
\pgfpathcurveto{\pgfqpoint{1.331cm}{1.3cm}}{\pgfqpoint{1.345cm}{1.335cm}}{\pgfqpoint{1.345cm}{1.371cm}}
\pgfusepath{fill}
\begin{pgfscope}
\pgfsetdash{}{0cm}
\pgfsetlinewidth{0.818mm}
\pgfsetroundcap
\pgfsetmiterlimit{4.0}
\pgfpathmoveto{\pgfqpoint{0.682cm}{0.671cm}}
\pgfpathlineto{\pgfqpoint{0.682cm}{0.042cm}}
\pgfusepath{stroke}
\end{pgfscope}
\end{pgfscope}
\end{pgfscope}
\end{pgfscope}
\end{tikzpicture}}}$ to be the stationary solution of
\[ \LL_{\varepsilon} X_{M, \varepsilon, >}^{\!\resizebox{0.6em}{!}{
\begin{tikzpicture}
\pgfpathmoveto{\pgfqpoint{0cm}{-0.035cm}}
\pgfpathlineto{\pgfqpoint{1.376cm}{-0.035cm}}
\pgfpathlineto{\pgfqpoint{1.376cm}{1.552cm}}
\pgfpathlineto{\pgfqpoint{0cm}{1.552cm}}
\pgfpathclose
\pgfusepath{clip}
\begin{pgfscope}
\begin{pgfscope}
\pgfpathmoveto{\pgfqpoint{0cm}{-0.035cm}}
\pgfpathlineto{\pgfqpoint{1.376cm}{-0.035cm}}
\pgfpathlineto{\pgfqpoint{1.376cm}{1.552cm}}
\pgfpathlineto{\pgfqpoint{0cm}{1.552cm}}
\pgfpathclose
\pgfusepath{clip}
\begin{pgfscope}
\begin{pgfscope}
\pgfsetdash{}{0cm}
\pgfsetlinewidth{0.818mm}
\pgfsetroundcap
\pgfsetroundjoin
\pgfsetmiterlimit{7.0}
\definecolor{eps2pgf_color}{gray}{0}\pgfsetstrokecolor{eps2pgf_color}\pgfsetfillcolor{eps2pgf_color}
\pgfpathmoveto{\pgfqpoint{0.117cm}{1.421cm}}
\pgfpathlineto{\pgfqpoint{0.682cm}{0.671cm}}
\pgfpathlineto{\pgfqpoint{1.246cm}{1.421cm}}
\pgfusepath{stroke}
\end{pgfscope}
\definecolor{eps2pgf_color}{gray}{0}\pgfsetstrokecolor{eps2pgf_color}\pgfsetfillcolor{eps2pgf_color}
\pgfpathmoveto{\pgfqpoint{0.273cm}{1.395cm}}
\pgfpathcurveto{\pgfqpoint{0.273cm}{1.432cm}}{\pgfqpoint{0.259cm}{1.467cm}}{\pgfqpoint{0.233cm}{1.492cm}}
\pgfpathcurveto{\pgfqpoint{0.207cm}{1.518cm}}{\pgfqpoint{0.173cm}{1.532cm}}{\pgfqpoint{0.137cm}{1.532cm}}
\pgfpathcurveto{\pgfqpoint{0.1cm}{1.532cm}}{\pgfqpoint{0.066cm}{1.518cm}}{\pgfqpoint{0.04cm}{1.492cm}}
\pgfpathcurveto{\pgfqpoint{0.014cm}{1.467cm}}{\pgfqpoint{0cm}{1.432cm}}{\pgfqpoint{0cm}{1.395cm}}
\pgfpathcurveto{\pgfqpoint{0cm}{1.359cm}}{\pgfqpoint{0.014cm}{1.324cm}}{\pgfqpoint{0.04cm}{1.299cm}}
\pgfpathcurveto{\pgfqpoint{0.066cm}{1.273cm}}{\pgfqpoint{0.1cm}{1.258cm}}{\pgfqpoint{0.137cm}{1.258cm}}
\pgfpathcurveto{\pgfqpoint{0.173cm}{1.258cm}}{\pgfqpoint{0.207cm}{1.273cm}}{\pgfqpoint{0.233cm}{1.299cm}}
\pgfpathcurveto{\pgfqpoint{0.259cm}{1.324cm}}{\pgfqpoint{0.273cm}{1.359cm}}{\pgfqpoint{0.273cm}{1.395cm}}
\pgfusepath{fill}
\begin{pgfscope}
\pgfsetdash{}{0cm}
\pgfsetlinewidth{0.818mm}
\pgfsetmiterlimit{7.0}
\pgfpathmoveto{\pgfqpoint{0.682cm}{0.671cm}}
\pgfpathlineto{\pgfqpoint{0.679cm}{1.418cm}}
\pgfusepath{stroke}
\end{pgfscope}
\pgfpathmoveto{\pgfqpoint{0.815cm}{1.399cm}}
\pgfpathcurveto{\pgfqpoint{0.815cm}{1.435cm}}{\pgfqpoint{0.801cm}{1.47cm}}{\pgfqpoint{0.775cm}{1.496cm}}
\pgfpathcurveto{\pgfqpoint{0.75cm}{1.521cm}}{\pgfqpoint{0.715cm}{1.536cm}}{\pgfqpoint{0.679cm}{1.536cm}}
\pgfpathcurveto{\pgfqpoint{0.643cm}{1.536cm}}{\pgfqpoint{0.608cm}{1.521cm}}{\pgfqpoint{0.582cm}{1.496cm}}
\pgfpathcurveto{\pgfqpoint{0.557cm}{1.47cm}}{\pgfqpoint{0.542cm}{1.435cm}}{\pgfqpoint{0.542cm}{1.399cm}}
\pgfpathcurveto{\pgfqpoint{0.542cm}{1.363cm}}{\pgfqpoint{0.557cm}{1.328cm}}{\pgfqpoint{0.582cm}{1.302cm}}
\pgfpathcurveto{\pgfqpoint{0.608cm}{1.276cm}}{\pgfqpoint{0.643cm}{1.262cm}}{\pgfqpoint{0.679cm}{1.262cm}}
\pgfpathcurveto{\pgfqpoint{0.715cm}{1.262cm}}{\pgfqpoint{0.75cm}{1.276cm}}{\pgfqpoint{0.775cm}{1.302cm}}
\pgfpathcurveto{\pgfqpoint{0.801cm}{1.328cm}}{\pgfqpoint{0.815cm}{1.363cm}}{\pgfqpoint{0.815cm}{1.399cm}}
\pgfusepath{fill}
\pgfpathmoveto{\pgfqpoint{1.345cm}{1.371cm}}
\pgfpathcurveto{\pgfqpoint{1.345cm}{1.408cm}}{\pgfqpoint{1.331cm}{1.442cm}}{\pgfqpoint{1.305cm}{1.468cm}}
\pgfpathcurveto{\pgfqpoint{1.28cm}{1.494cm}}{\pgfqpoint{1.245cm}{1.508cm}}{\pgfqpoint{1.209cm}{1.508cm}}
\pgfpathcurveto{\pgfqpoint{1.172cm}{1.508cm}}{\pgfqpoint{1.138cm}{1.494cm}}{\pgfqpoint{1.112cm}{1.468cm}}
\pgfpathcurveto{\pgfqpoint{1.087cm}{1.442cm}}{\pgfqpoint{1.072cm}{1.408cm}}{\pgfqpoint{1.072cm}{1.371cm}}
\pgfpathcurveto{\pgfqpoint{1.072cm}{1.335cm}}{\pgfqpoint{1.087cm}{1.3cm}}{\pgfqpoint{1.112cm}{1.274cm}}
\pgfpathcurveto{\pgfqpoint{1.138cm}{1.249cm}}{\pgfqpoint{1.172cm}{1.234cm}}{\pgfqpoint{1.209cm}{1.234cm}}
\pgfpathcurveto{\pgfqpoint{1.245cm}{1.234cm}}{\pgfqpoint{1.28cm}{1.249cm}}{\pgfqpoint{1.305cm}{1.274cm}}
\pgfpathcurveto{\pgfqpoint{1.331cm}{1.3cm}}{\pgfqpoint{1.345cm}{1.335cm}}{\pgfqpoint{1.345cm}{1.371cm}}
\pgfusepath{fill}
\begin{pgfscope}
\pgfsetdash{}{0cm}
\pgfsetlinewidth{0.818mm}
\pgfsetroundcap
\pgfsetmiterlimit{4.0}
\pgfpathmoveto{\pgfqpoint{0.682cm}{0.671cm}}
\pgfpathlineto{\pgfqpoint{0.682cm}{0.042cm}}
\pgfusepath{stroke}
\end{pgfscope}
\end{pgfscope}
\end{pgfscope}
\end{pgfscope}
\end{tikzpicture}}} = \llbracket X_{M,
   \varepsilon}^3 \rrbracket - \llbracket X_{M, \varepsilon}^3
   \rrbracket_{\leqslant}. \]
By choosing $K$ we can have that
\[ \| X_{M, \varepsilon, >}^{\!\resizebox{0.6em}{!}{
\begin{tikzpicture}
\pgfpathmoveto{\pgfqpoint{0cm}{-0.035cm}}
\pgfpathlineto{\pgfqpoint{1.376cm}{-0.035cm}}
\pgfpathlineto{\pgfqpoint{1.376cm}{1.552cm}}
\pgfpathlineto{\pgfqpoint{0cm}{1.552cm}}
\pgfpathclose
\pgfusepath{clip}
\begin{pgfscope}
\begin{pgfscope}
\pgfpathmoveto{\pgfqpoint{0cm}{-0.035cm}}
\pgfpathlineto{\pgfqpoint{1.376cm}{-0.035cm}}
\pgfpathlineto{\pgfqpoint{1.376cm}{1.552cm}}
\pgfpathlineto{\pgfqpoint{0cm}{1.552cm}}
\pgfpathclose
\pgfusepath{clip}
\begin{pgfscope}
\begin{pgfscope}
\pgfsetdash{}{0cm}
\pgfsetlinewidth{0.818mm}
\pgfsetroundcap
\pgfsetroundjoin
\pgfsetmiterlimit{7.0}
\definecolor{eps2pgf_color}{gray}{0}\pgfsetstrokecolor{eps2pgf_color}\pgfsetfillcolor{eps2pgf_color}
\pgfpathmoveto{\pgfqpoint{0.117cm}{1.421cm}}
\pgfpathlineto{\pgfqpoint{0.682cm}{0.671cm}}
\pgfpathlineto{\pgfqpoint{1.246cm}{1.421cm}}
\pgfusepath{stroke}
\end{pgfscope}
\definecolor{eps2pgf_color}{gray}{0}\pgfsetstrokecolor{eps2pgf_color}\pgfsetfillcolor{eps2pgf_color}
\pgfpathmoveto{\pgfqpoint{0.273cm}{1.395cm}}
\pgfpathcurveto{\pgfqpoint{0.273cm}{1.432cm}}{\pgfqpoint{0.259cm}{1.467cm}}{\pgfqpoint{0.233cm}{1.492cm}}
\pgfpathcurveto{\pgfqpoint{0.207cm}{1.518cm}}{\pgfqpoint{0.173cm}{1.532cm}}{\pgfqpoint{0.137cm}{1.532cm}}
\pgfpathcurveto{\pgfqpoint{0.1cm}{1.532cm}}{\pgfqpoint{0.066cm}{1.518cm}}{\pgfqpoint{0.04cm}{1.492cm}}
\pgfpathcurveto{\pgfqpoint{0.014cm}{1.467cm}}{\pgfqpoint{0cm}{1.432cm}}{\pgfqpoint{0cm}{1.395cm}}
\pgfpathcurveto{\pgfqpoint{0cm}{1.359cm}}{\pgfqpoint{0.014cm}{1.324cm}}{\pgfqpoint{0.04cm}{1.299cm}}
\pgfpathcurveto{\pgfqpoint{0.066cm}{1.273cm}}{\pgfqpoint{0.1cm}{1.258cm}}{\pgfqpoint{0.137cm}{1.258cm}}
\pgfpathcurveto{\pgfqpoint{0.173cm}{1.258cm}}{\pgfqpoint{0.207cm}{1.273cm}}{\pgfqpoint{0.233cm}{1.299cm}}
\pgfpathcurveto{\pgfqpoint{0.259cm}{1.324cm}}{\pgfqpoint{0.273cm}{1.359cm}}{\pgfqpoint{0.273cm}{1.395cm}}
\pgfusepath{fill}
\begin{pgfscope}
\pgfsetdash{}{0cm}
\pgfsetlinewidth{0.818mm}
\pgfsetmiterlimit{7.0}
\pgfpathmoveto{\pgfqpoint{0.682cm}{0.671cm}}
\pgfpathlineto{\pgfqpoint{0.679cm}{1.418cm}}
\pgfusepath{stroke}
\end{pgfscope}
\pgfpathmoveto{\pgfqpoint{0.815cm}{1.399cm}}
\pgfpathcurveto{\pgfqpoint{0.815cm}{1.435cm}}{\pgfqpoint{0.801cm}{1.47cm}}{\pgfqpoint{0.775cm}{1.496cm}}
\pgfpathcurveto{\pgfqpoint{0.75cm}{1.521cm}}{\pgfqpoint{0.715cm}{1.536cm}}{\pgfqpoint{0.679cm}{1.536cm}}
\pgfpathcurveto{\pgfqpoint{0.643cm}{1.536cm}}{\pgfqpoint{0.608cm}{1.521cm}}{\pgfqpoint{0.582cm}{1.496cm}}
\pgfpathcurveto{\pgfqpoint{0.557cm}{1.47cm}}{\pgfqpoint{0.542cm}{1.435cm}}{\pgfqpoint{0.542cm}{1.399cm}}
\pgfpathcurveto{\pgfqpoint{0.542cm}{1.363cm}}{\pgfqpoint{0.557cm}{1.328cm}}{\pgfqpoint{0.582cm}{1.302cm}}
\pgfpathcurveto{\pgfqpoint{0.608cm}{1.276cm}}{\pgfqpoint{0.643cm}{1.262cm}}{\pgfqpoint{0.679cm}{1.262cm}}
\pgfpathcurveto{\pgfqpoint{0.715cm}{1.262cm}}{\pgfqpoint{0.75cm}{1.276cm}}{\pgfqpoint{0.775cm}{1.302cm}}
\pgfpathcurveto{\pgfqpoint{0.801cm}{1.328cm}}{\pgfqpoint{0.815cm}{1.363cm}}{\pgfqpoint{0.815cm}{1.399cm}}
\pgfusepath{fill}
\pgfpathmoveto{\pgfqpoint{1.345cm}{1.371cm}}
\pgfpathcurveto{\pgfqpoint{1.345cm}{1.408cm}}{\pgfqpoint{1.331cm}{1.442cm}}{\pgfqpoint{1.305cm}{1.468cm}}
\pgfpathcurveto{\pgfqpoint{1.28cm}{1.494cm}}{\pgfqpoint{1.245cm}{1.508cm}}{\pgfqpoint{1.209cm}{1.508cm}}
\pgfpathcurveto{\pgfqpoint{1.172cm}{1.508cm}}{\pgfqpoint{1.138cm}{1.494cm}}{\pgfqpoint{1.112cm}{1.468cm}}
\pgfpathcurveto{\pgfqpoint{1.087cm}{1.442cm}}{\pgfqpoint{1.072cm}{1.408cm}}{\pgfqpoint{1.072cm}{1.371cm}}
\pgfpathcurveto{\pgfqpoint{1.072cm}{1.335cm}}{\pgfqpoint{1.087cm}{1.3cm}}{\pgfqpoint{1.112cm}{1.274cm}}
\pgfpathcurveto{\pgfqpoint{1.138cm}{1.249cm}}{\pgfqpoint{1.172cm}{1.234cm}}{\pgfqpoint{1.209cm}{1.234cm}}
\pgfpathcurveto{\pgfqpoint{1.245cm}{1.234cm}}{\pgfqpoint{1.28cm}{1.249cm}}{\pgfqpoint{1.305cm}{1.274cm}}
\pgfpathcurveto{\pgfqpoint{1.331cm}{1.3cm}}{\pgfqpoint{1.345cm}{1.335cm}}{\pgfqpoint{1.345cm}{1.371cm}}
\pgfusepath{fill}
\begin{pgfscope}
\pgfsetdash{}{0cm}
\pgfsetlinewidth{0.818mm}
\pgfsetroundcap
\pgfsetmiterlimit{4.0}
\pgfpathmoveto{\pgfqpoint{0.682cm}{0.671cm}}
\pgfpathlineto{\pgfqpoint{0.682cm}{0.042cm}}
\pgfusepath{stroke}
\end{pgfscope}
\end{pgfscope}
\end{pgfscope}
\end{pgfscope}
\end{tikzpicture}}} \|_{C_T L^{\infty, \varepsilon}
   (\rho^{\sigma})} \lesssim 2^{- K (1 / 2 - 2 \kappa)} \| X_{M,
   \varepsilon, >}^{\!\resizebox{0.6em}{!}{
\begin{tikzpicture}
\pgfpathmoveto{\pgfqpoint{0cm}{-0.035cm}}
\pgfpathlineto{\pgfqpoint{1.376cm}{-0.035cm}}
\pgfpathlineto{\pgfqpoint{1.376cm}{1.552cm}}
\pgfpathlineto{\pgfqpoint{0cm}{1.552cm}}
\pgfpathclose
\pgfusepath{clip}
\begin{pgfscope}
\begin{pgfscope}
\pgfpathmoveto{\pgfqpoint{0cm}{-0.035cm}}
\pgfpathlineto{\pgfqpoint{1.376cm}{-0.035cm}}
\pgfpathlineto{\pgfqpoint{1.376cm}{1.552cm}}
\pgfpathlineto{\pgfqpoint{0cm}{1.552cm}}
\pgfpathclose
\pgfusepath{clip}
\begin{pgfscope}
\begin{pgfscope}
\pgfsetdash{}{0cm}
\pgfsetlinewidth{0.818mm}
\pgfsetroundcap
\pgfsetroundjoin
\pgfsetmiterlimit{7.0}
\definecolor{eps2pgf_color}{gray}{0}\pgfsetstrokecolor{eps2pgf_color}\pgfsetfillcolor{eps2pgf_color}
\pgfpathmoveto{\pgfqpoint{0.117cm}{1.421cm}}
\pgfpathlineto{\pgfqpoint{0.682cm}{0.671cm}}
\pgfpathlineto{\pgfqpoint{1.246cm}{1.421cm}}
\pgfusepath{stroke}
\end{pgfscope}
\definecolor{eps2pgf_color}{gray}{0}\pgfsetstrokecolor{eps2pgf_color}\pgfsetfillcolor{eps2pgf_color}
\pgfpathmoveto{\pgfqpoint{0.273cm}{1.395cm}}
\pgfpathcurveto{\pgfqpoint{0.273cm}{1.432cm}}{\pgfqpoint{0.259cm}{1.467cm}}{\pgfqpoint{0.233cm}{1.492cm}}
\pgfpathcurveto{\pgfqpoint{0.207cm}{1.518cm}}{\pgfqpoint{0.173cm}{1.532cm}}{\pgfqpoint{0.137cm}{1.532cm}}
\pgfpathcurveto{\pgfqpoint{0.1cm}{1.532cm}}{\pgfqpoint{0.066cm}{1.518cm}}{\pgfqpoint{0.04cm}{1.492cm}}
\pgfpathcurveto{\pgfqpoint{0.014cm}{1.467cm}}{\pgfqpoint{0cm}{1.432cm}}{\pgfqpoint{0cm}{1.395cm}}
\pgfpathcurveto{\pgfqpoint{0cm}{1.359cm}}{\pgfqpoint{0.014cm}{1.324cm}}{\pgfqpoint{0.04cm}{1.299cm}}
\pgfpathcurveto{\pgfqpoint{0.066cm}{1.273cm}}{\pgfqpoint{0.1cm}{1.258cm}}{\pgfqpoint{0.137cm}{1.258cm}}
\pgfpathcurveto{\pgfqpoint{0.173cm}{1.258cm}}{\pgfqpoint{0.207cm}{1.273cm}}{\pgfqpoint{0.233cm}{1.299cm}}
\pgfpathcurveto{\pgfqpoint{0.259cm}{1.324cm}}{\pgfqpoint{0.273cm}{1.359cm}}{\pgfqpoint{0.273cm}{1.395cm}}
\pgfusepath{fill}
\begin{pgfscope}
\pgfsetdash{}{0cm}
\pgfsetlinewidth{0.818mm}
\pgfsetmiterlimit{7.0}
\pgfpathmoveto{\pgfqpoint{0.682cm}{0.671cm}}
\pgfpathlineto{\pgfqpoint{0.679cm}{1.418cm}}
\pgfusepath{stroke}
\end{pgfscope}
\pgfpathmoveto{\pgfqpoint{0.815cm}{1.399cm}}
\pgfpathcurveto{\pgfqpoint{0.815cm}{1.435cm}}{\pgfqpoint{0.801cm}{1.47cm}}{\pgfqpoint{0.775cm}{1.496cm}}
\pgfpathcurveto{\pgfqpoint{0.75cm}{1.521cm}}{\pgfqpoint{0.715cm}{1.536cm}}{\pgfqpoint{0.679cm}{1.536cm}}
\pgfpathcurveto{\pgfqpoint{0.643cm}{1.536cm}}{\pgfqpoint{0.608cm}{1.521cm}}{\pgfqpoint{0.582cm}{1.496cm}}
\pgfpathcurveto{\pgfqpoint{0.557cm}{1.47cm}}{\pgfqpoint{0.542cm}{1.435cm}}{\pgfqpoint{0.542cm}{1.399cm}}
\pgfpathcurveto{\pgfqpoint{0.542cm}{1.363cm}}{\pgfqpoint{0.557cm}{1.328cm}}{\pgfqpoint{0.582cm}{1.302cm}}
\pgfpathcurveto{\pgfqpoint{0.608cm}{1.276cm}}{\pgfqpoint{0.643cm}{1.262cm}}{\pgfqpoint{0.679cm}{1.262cm}}
\pgfpathcurveto{\pgfqpoint{0.715cm}{1.262cm}}{\pgfqpoint{0.75cm}{1.276cm}}{\pgfqpoint{0.775cm}{1.302cm}}
\pgfpathcurveto{\pgfqpoint{0.801cm}{1.328cm}}{\pgfqpoint{0.815cm}{1.363cm}}{\pgfqpoint{0.815cm}{1.399cm}}
\pgfusepath{fill}
\pgfpathmoveto{\pgfqpoint{1.345cm}{1.371cm}}
\pgfpathcurveto{\pgfqpoint{1.345cm}{1.408cm}}{\pgfqpoint{1.331cm}{1.442cm}}{\pgfqpoint{1.305cm}{1.468cm}}
\pgfpathcurveto{\pgfqpoint{1.28cm}{1.494cm}}{\pgfqpoint{1.245cm}{1.508cm}}{\pgfqpoint{1.209cm}{1.508cm}}
\pgfpathcurveto{\pgfqpoint{1.172cm}{1.508cm}}{\pgfqpoint{1.138cm}{1.494cm}}{\pgfqpoint{1.112cm}{1.468cm}}
\pgfpathcurveto{\pgfqpoint{1.087cm}{1.442cm}}{\pgfqpoint{1.072cm}{1.408cm}}{\pgfqpoint{1.072cm}{1.371cm}}
\pgfpathcurveto{\pgfqpoint{1.072cm}{1.335cm}}{\pgfqpoint{1.087cm}{1.3cm}}{\pgfqpoint{1.112cm}{1.274cm}}
\pgfpathcurveto{\pgfqpoint{1.138cm}{1.249cm}}{\pgfqpoint{1.172cm}{1.234cm}}{\pgfqpoint{1.209cm}{1.234cm}}
\pgfpathcurveto{\pgfqpoint{1.245cm}{1.234cm}}{\pgfqpoint{1.28cm}{1.249cm}}{\pgfqpoint{1.305cm}{1.274cm}}
\pgfpathcurveto{\pgfqpoint{1.331cm}{1.3cm}}{\pgfqpoint{1.345cm}{1.335cm}}{\pgfqpoint{1.345cm}{1.371cm}}
\pgfusepath{fill}
\begin{pgfscope}
\pgfsetdash{}{0cm}
\pgfsetlinewidth{0.818mm}
\pgfsetroundcap
\pgfsetmiterlimit{4.0}
\pgfpathmoveto{\pgfqpoint{0.682cm}{0.671cm}}
\pgfpathlineto{\pgfqpoint{0.682cm}{0.042cm}}
\pgfusepath{stroke}
\end{pgfscope}
\end{pgfscope}
\end{pgfscope}
\end{pgfscope}
\end{tikzpicture}}} \|_{C_T \mathscr{C} \hspace{.1em}^{1 / 2 - \kappa,
   \varepsilon} (\rho^{\sigma})} \lesssim 2^{- K (1 / 2 - 2 \kappa)} \|
   \mathbb{X}_{M, \varepsilon} \|^3 \lesssim \| \mathbb{X}_{M, \varepsilon}
   \|^2 \]
which holds true provided
\[ 2^{K / 2} = \| \mathbb{X}_{M, \varepsilon} \|^{1 / (1 - 4 \kappa)} .
\]
Next, we redefine $Y_{M, \varepsilon}$ to solve
\[  Y_{M, \varepsilon} = - \lambda X_{M,
   \varepsilon, >}^{\!\resizebox{0.6em}{!}{
\begin{tikzpicture}
\pgfpathmoveto{\pgfqpoint{0cm}{-0.035cm}}
\pgfpathlineto{\pgfqpoint{1.376cm}{-0.035cm}}
\pgfpathlineto{\pgfqpoint{1.376cm}{1.552cm}}
\pgfpathlineto{\pgfqpoint{0cm}{1.552cm}}
\pgfpathclose
\pgfusepath{clip}
\begin{pgfscope}
\begin{pgfscope}
\pgfpathmoveto{\pgfqpoint{0cm}{-0.035cm}}
\pgfpathlineto{\pgfqpoint{1.376cm}{-0.035cm}}
\pgfpathlineto{\pgfqpoint{1.376cm}{1.552cm}}
\pgfpathlineto{\pgfqpoint{0cm}{1.552cm}}
\pgfpathclose
\pgfusepath{clip}
\begin{pgfscope}
\begin{pgfscope}
\pgfsetdash{}{0cm}
\pgfsetlinewidth{0.818mm}
\pgfsetroundcap
\pgfsetroundjoin
\pgfsetmiterlimit{7.0}
\definecolor{eps2pgf_color}{gray}{0}\pgfsetstrokecolor{eps2pgf_color}\pgfsetfillcolor{eps2pgf_color}
\pgfpathmoveto{\pgfqpoint{0.117cm}{1.421cm}}
\pgfpathlineto{\pgfqpoint{0.682cm}{0.671cm}}
\pgfpathlineto{\pgfqpoint{1.246cm}{1.421cm}}
\pgfusepath{stroke}
\end{pgfscope}
\definecolor{eps2pgf_color}{gray}{0}\pgfsetstrokecolor{eps2pgf_color}\pgfsetfillcolor{eps2pgf_color}
\pgfpathmoveto{\pgfqpoint{0.273cm}{1.395cm}}
\pgfpathcurveto{\pgfqpoint{0.273cm}{1.432cm}}{\pgfqpoint{0.259cm}{1.467cm}}{\pgfqpoint{0.233cm}{1.492cm}}
\pgfpathcurveto{\pgfqpoint{0.207cm}{1.518cm}}{\pgfqpoint{0.173cm}{1.532cm}}{\pgfqpoint{0.137cm}{1.532cm}}
\pgfpathcurveto{\pgfqpoint{0.1cm}{1.532cm}}{\pgfqpoint{0.066cm}{1.518cm}}{\pgfqpoint{0.04cm}{1.492cm}}
\pgfpathcurveto{\pgfqpoint{0.014cm}{1.467cm}}{\pgfqpoint{0cm}{1.432cm}}{\pgfqpoint{0cm}{1.395cm}}
\pgfpathcurveto{\pgfqpoint{0cm}{1.359cm}}{\pgfqpoint{0.014cm}{1.324cm}}{\pgfqpoint{0.04cm}{1.299cm}}
\pgfpathcurveto{\pgfqpoint{0.066cm}{1.273cm}}{\pgfqpoint{0.1cm}{1.258cm}}{\pgfqpoint{0.137cm}{1.258cm}}
\pgfpathcurveto{\pgfqpoint{0.173cm}{1.258cm}}{\pgfqpoint{0.207cm}{1.273cm}}{\pgfqpoint{0.233cm}{1.299cm}}
\pgfpathcurveto{\pgfqpoint{0.259cm}{1.324cm}}{\pgfqpoint{0.273cm}{1.359cm}}{\pgfqpoint{0.273cm}{1.395cm}}
\pgfusepath{fill}
\begin{pgfscope}
\pgfsetdash{}{0cm}
\pgfsetlinewidth{0.818mm}
\pgfsetmiterlimit{7.0}
\pgfpathmoveto{\pgfqpoint{0.682cm}{0.671cm}}
\pgfpathlineto{\pgfqpoint{0.679cm}{1.418cm}}
\pgfusepath{stroke}
\end{pgfscope}
\pgfpathmoveto{\pgfqpoint{0.815cm}{1.399cm}}
\pgfpathcurveto{\pgfqpoint{0.815cm}{1.435cm}}{\pgfqpoint{0.801cm}{1.47cm}}{\pgfqpoint{0.775cm}{1.496cm}}
\pgfpathcurveto{\pgfqpoint{0.75cm}{1.521cm}}{\pgfqpoint{0.715cm}{1.536cm}}{\pgfqpoint{0.679cm}{1.536cm}}
\pgfpathcurveto{\pgfqpoint{0.643cm}{1.536cm}}{\pgfqpoint{0.608cm}{1.521cm}}{\pgfqpoint{0.582cm}{1.496cm}}
\pgfpathcurveto{\pgfqpoint{0.557cm}{1.47cm}}{\pgfqpoint{0.542cm}{1.435cm}}{\pgfqpoint{0.542cm}{1.399cm}}
\pgfpathcurveto{\pgfqpoint{0.542cm}{1.363cm}}{\pgfqpoint{0.557cm}{1.328cm}}{\pgfqpoint{0.582cm}{1.302cm}}
\pgfpathcurveto{\pgfqpoint{0.608cm}{1.276cm}}{\pgfqpoint{0.643cm}{1.262cm}}{\pgfqpoint{0.679cm}{1.262cm}}
\pgfpathcurveto{\pgfqpoint{0.715cm}{1.262cm}}{\pgfqpoint{0.75cm}{1.276cm}}{\pgfqpoint{0.775cm}{1.302cm}}
\pgfpathcurveto{\pgfqpoint{0.801cm}{1.328cm}}{\pgfqpoint{0.815cm}{1.363cm}}{\pgfqpoint{0.815cm}{1.399cm}}
\pgfusepath{fill}
\pgfpathmoveto{\pgfqpoint{1.345cm}{1.371cm}}
\pgfpathcurveto{\pgfqpoint{1.345cm}{1.408cm}}{\pgfqpoint{1.331cm}{1.442cm}}{\pgfqpoint{1.305cm}{1.468cm}}
\pgfpathcurveto{\pgfqpoint{1.28cm}{1.494cm}}{\pgfqpoint{1.245cm}{1.508cm}}{\pgfqpoint{1.209cm}{1.508cm}}
\pgfpathcurveto{\pgfqpoint{1.172cm}{1.508cm}}{\pgfqpoint{1.138cm}{1.494cm}}{\pgfqpoint{1.112cm}{1.468cm}}
\pgfpathcurveto{\pgfqpoint{1.087cm}{1.442cm}}{\pgfqpoint{1.072cm}{1.408cm}}{\pgfqpoint{1.072cm}{1.371cm}}
\pgfpathcurveto{\pgfqpoint{1.072cm}{1.335cm}}{\pgfqpoint{1.087cm}{1.3cm}}{\pgfqpoint{1.112cm}{1.274cm}}
\pgfpathcurveto{\pgfqpoint{1.138cm}{1.249cm}}{\pgfqpoint{1.172cm}{1.234cm}}{\pgfqpoint{1.209cm}{1.234cm}}
\pgfpathcurveto{\pgfqpoint{1.245cm}{1.234cm}}{\pgfqpoint{1.28cm}{1.249cm}}{\pgfqpoint{1.305cm}{1.274cm}}
\pgfpathcurveto{\pgfqpoint{1.331cm}{1.3cm}}{\pgfqpoint{1.345cm}{1.335cm}}{\pgfqpoint{1.345cm}{1.371cm}}
\pgfusepath{fill}
\begin{pgfscope}
\pgfsetdash{}{0cm}
\pgfsetlinewidth{0.818mm}
\pgfsetroundcap
\pgfsetmiterlimit{4.0}
\pgfpathmoveto{\pgfqpoint{0.682cm}{0.671cm}}
\pgfpathlineto{\pgfqpoint{0.682cm}{0.042cm}}
\pgfusepath{stroke}
\end{pgfscope}
\end{pgfscope}
\end{pgfscope}
\end{pgfscope}
\end{tikzpicture}}} - \LL_{\varepsilon}^{- 1} [3 \lambda
   (\UU^{\varepsilon}_{>} \llbracket X_{M, \varepsilon}^2 \rrbracket) \succ
   Y_{M, \varepsilon}] . \]
The estimates of Lemma~\ref{lem:Y1} are still valid with obvious
modifications. In addition, we obtain
\[ \| \rho^{\sigma} Y_{M, \varepsilon} \|_{C_T L^{\infty, \varepsilon}
   (\rho^{\sigma})} \lesssim \lambda \| \mathbb{X}_{M, \varepsilon} \|^2,
   \qquad \| \rho^{\sigma} Y_{M, \varepsilon} \|_{C_T \mathscr{C}
   \hspace{.1em}^{1 / 2 - \kappa, \varepsilon} (\rho^{\sigma})} \lesssim
   \lambda \| \mathbb{X}_{M, \varepsilon} \|^3, \]
and by interpolation it follows for $a \in [0, 1 / 2 - \kappa]$ that
\begin{equation}
  \| \rho^{\sigma} Y_{M, \varepsilon} \|_{C_T \mathscr{C} \hspace{.1em}^{a,
  \varepsilon} (\rho^{\sigma})} \lesssim \lambda \| \mathbb{X}_{M,
  \varepsilon} \|^{2 + a / (1 / 2 - \kappa)} . \label{eq:interp-Y}
\end{equation}
From now on we avoid, as usual, to specify explicitly the dependence on $M$
since it does not play any role in the estimates. The energy
equality~{\eqref{eq:en12}} in Lemma~\ref{lem:energy12} now reads
\begin{equation}
  \frac{1}{2} \partial_t \| \rho^2 \phi_{\varepsilon} \|_{L^{2,
  \varepsilon}}^2 + \Upsilon_{\varepsilon} = \Theta_{\rho^4, \varepsilon} +
  \Psi_{\rho^4, \varepsilon} + \langle \rho^4 \phi_{\varepsilon}, - \lambda
  \llbracket X_{\varepsilon}^3 \rrbracket_{\leqslant} \rangle_{\varepsilon},
  \label{eq:en12-int}
\end{equation}
where
\[
\Upsilon_{\varepsilon}:=\lambda\|\rho\phi_{\varepsilon}\|_{L^{4,\varepsilon}}^{4}+m^{2}\|\rho^{2}\psi_{\varepsilon}\|_{L^{2,\varepsilon}}^{2}+\|\rho^{2}\nabla_{\varepsilon}\psi_{\varepsilon}\|_{L^{2,\varepsilon}}^{2}
\]
and $\Theta_{\rho^4, \varepsilon}, \Psi_{\rho^4, \varepsilon}$ where defined in Lemma~\ref{lem:energy12}.
Our goal is to bound the right hand side of \eqref{eq:en12-int}
with no more than a factor $\| \mathbb{X}_{M, \varepsilon} \|^{8 + \vartheta}$
for some $\vartheta = O (\kappa)$. In view of the estimates within the proof of Lemma \ref{lemma:bounds-rhs1} we
observe that the bounds {\eqref{eq:XY}}, {\eqref{eq:Y3}}, {\eqref{eq:Y2}},
{\eqref{eq:Y11}}, {\eqref{eq:U0}} and {\eqref{eq:Z0}} need to be improved.

\begin{lemma}
  \label{lemma:bounds-rhs1-int}Let $\rho$ be a weight such that $\rho^{\iota}
  \in L^{4, 0}$ for some $\iota \in (0, 1)$. Then there is $\vartheta = O
  (\kappa)>0$ such that
  \[ | \Theta_{\rho^4, \varepsilon} | + | \Psi_{\rho^4, \varepsilon} |+|\langle \rho^4 \phi_{\varepsilon}, - \lambda
\llbracket X_{\varepsilon}^3 \rrbracket_{\leqslant} \rangle_{\varepsilon}|
     \leqslant C_{\delta}  (\lambda + \lambda^{7 / 3} | \log t|^{4 / 3} +
     \lambda^5) \| \mathbb{X}_{\varepsilon} \|^{8 + \vartheta} + \delta
     \Upsilon_{\varepsilon} . \]
\end{lemma}

\begin{proof}
  Let us begin with a new bound for the term with $X_{\varepsilon}
  Y_{\varepsilon}^2$ appearing in {\eqref{eq:XY}}. For the resonant term we
  get from the interpolation estimate~{\eqref{eq:interp-Y}} that the
  bound~{\eqref{eq:XY2-res}} can be updated as
  \[ \| \rho^{\sigma} X_{\varepsilon} \circ Y_{\varepsilon}^2 \|_{C_T
     \mathscr{C} \hspace{.1em}^{- \kappa, \varepsilon}} \lesssim \lambda^2 \|
     \mathbb{X}_{\varepsilon} \|^{6 + \vartheta} + \lambda^3 \|
     \mathbb{X}_{\varepsilon} \|^{5 + \vartheta} \lesssim (\lambda^2 +
     \lambda^3) \| \mathbb{X}_{\varepsilon} \|^{6 + \vartheta} \]
  where we used that, due to the presence of the localizer (see~{\eqref{eq:U11}}), we can bound
  \begin{equation}
    \left\| \rho^{\sigma} \UU_{>} \llbracket X_{\varepsilon}^2 \rrbracket
    \right\|_{\mathscr{C} \hspace{.1em}^{- 3 / 2 + 2 \kappa, \varepsilon}}
    \lesssim \| \rho^{\sigma} \llbracket X_{\varepsilon}^2 \rrbracket
    \|_{\mathscr{C} \hspace{.1em}^{- 1 - \kappa, \varepsilon}} \left( 1 +
    \lambda \| \rho^{\sigma} \llbracket X_{\varepsilon}^2 \rrbracket
    \|_{\mathscr{C} \hspace{.1em}^{- 1 - \kappa, \varepsilon}} \right)^{- (1 -
    6 \kappa)} \lesssim \| \mathbb{X}_{\varepsilon} \|^{\vartheta}
    \label{eq:impro-X2}
  \end{equation}
  giving an improved bound for the paracontrolled term which reads as follows
  \[ \left\| \rho^{4 \sigma} X_{\varepsilon} \circ \left( 2 Y_{\varepsilon}
     \prec \LL_{\varepsilon}^{- 1} \left[ 3 \lambda \left( \UU_{>} \llbracket
     X_{\varepsilon}^2 \rrbracket \right) \succ Y_{\varepsilon} \right]
     \right) \right\|_{\mathscr{C} \hspace{.1em}^{- \kappa, \varepsilon}} \]
  \[ \lesssim \lambda \| \rho^{\sigma} X_{\varepsilon} \|_{\mathscr{C}
     \hspace{.1em}^{- 1 / 2 - \kappa, \varepsilon}} \| \rho^{\sigma}
     Y_{\varepsilon} \|_{L^{\infty, \varepsilon}}^2 \left\| \rho^{\sigma}
     \UU_{>} \llbracket X_{\varepsilon}^2 \rrbracket \right\|_{\mathscr{C}
     \hspace{.1em}^{- 3 / 2 + 2 \kappa, \varepsilon}} \lesssim \lambda^3 \|
     \mathbb{X}_{\varepsilon} \|^{5 + \vartheta} . \]
  Consequently, for $\theta = \frac{1 - 4 \kappa}{1 - 2 \kappa}$
  \[ \lambda | \langle \rho^4 \phi_{\varepsilon}, X_{\varepsilon} \circ
     Y_{\varepsilon}^2 \rangle_{\varepsilon} | \lesssim \lambda \|
     \rho^{\sigma} X_{\varepsilon} \circ Y_{\varepsilon}^2 \|_{\mathscr{C}
     \hspace{.1em}^{- \kappa, \varepsilon}} \| \rho^{4 - \sigma}
     \phi_{\varepsilon} \|_{B^{\kappa, \varepsilon}_{1, 1}} \lesssim
     (\lambda^3 + \lambda^4) \| \mathbb{X}_{\varepsilon} \|^{6 + \vartheta} \|
     \rho \phi_{\varepsilon} \|^{\theta}_{L^{4, \varepsilon}} \| \rho^2
     \phi_{\varepsilon} \|_{H^{1 - 2 \kappa, \varepsilon}}^{1 - \theta} \]
  \[ \leqslant (\lambda^{(12 - \theta) / (2 + \theta)} + \lambda^{(16 -
     \theta) / (2 + \theta)}) C_{\delta} \| \mathbb{X}_{\varepsilon} \|^{8 +
     \vartheta} + \delta \Upsilon_{\varepsilon} . \]
  For the paraproducts we have for $\theta = \frac{1 / 2 - 4 \kappa}{1 - 2
  \kappa}$
  \[ \lambda | \langle \rho^4 \phi_{\varepsilon}, X_{\varepsilon} \Join
     Y_{\varepsilon}^2 \rangle_{\varepsilon} | \lesssim \lambda \| \rho^{4 - 2
     \sigma} \phi_{\varepsilon} \|_{B^{1 / 2 + \kappa,\varepsilon}_{1, 1}} \|
     \rho^{\sigma} X_{\varepsilon} \|_{\mathscr{C} \hspace{.1em}^{- 1 / 2 -
     \kappa, \varepsilon}} \| \rho^{\sigma} Y_{\varepsilon} \|_{L^{\infty,
     \varepsilon}}^2 \]
  \[ \lesssim \lambda^3 \| \mathbb{X}_{\varepsilon} \|^5 \| \rho
     \phi_{\varepsilon} \|^{\theta}_{L^{4, \varepsilon}} \| \rho^2
     \phi_{\varepsilon} \|^{1 - \theta}_{H^{1 - 2 \kappa,\varepsilon}} \leqslant
     \lambda^{(12 - \theta) / (2 + \theta)} C_{\delta} \|
     \mathbb{X}_{\varepsilon} \|^8 + \delta \Upsilon_{\varepsilon} . \]
  Let us now consider the term with $X_{\varepsilon} Y_{\varepsilon}$ always
  in {\eqref{eq:XY}}. In view of {\eqref{eq:XY-res}}, {\eqref{eq:XY-par1}},
  {\eqref{eq:XY-par2}} we shall modify the bound of the resonant
  product which using the decomposition {\eqref{eq:XY-res3}} together with
  {\eqref{eq:XY-res}} and the bound~{\eqref{eq:impro-X2}}. We obtain
  \[ \| \rho^{\sigma} X_{\varepsilon} \circ Y_{\varepsilon} \|_{\mathscr{C}
     \hspace{.1em}^{- \kappa, \varepsilon}} \lesssim \lambda \|
     \mathbb{X}_{\varepsilon} \|^4 + \lambda^2 \| \mathbb{X}_{\varepsilon}
     \|^{3 + \vartheta} \lesssim (\lambda + \lambda^2) \|
     \mathbb{X}_{\varepsilon} \|^4 , \]
  and consequently, for $\theta = \frac{1 - 4 \kappa}{1 - 2 \kappa}$,
  \[ \lambda | \langle \rho^4 \phi_{\varepsilon}^2, X_{\varepsilon} \circ
     Y_{\varepsilon} \rangle_{\varepsilon} | \lesssim \lambda \| \rho^{\sigma}
     X_{\varepsilon} \circ Y_{\varepsilon} \|_{\mathscr{C} \hspace{.1em}^{-
     \kappa, \varepsilon}} \| \rho^{4 - \sigma} \phi_{\varepsilon}^2
     \|_{B^{\kappa, \varepsilon}_{1, 1}} \lesssim (\lambda^2 + \lambda^3) \|
     \mathbb{X}_{\varepsilon} \|^4 \| \rho \phi_{\varepsilon} \|^{1 +
     \theta}_{L^{4, \varepsilon}} \| \rho^2 \phi_{\varepsilon} \|_{H^{1 - 2
     \kappa, \varepsilon}}^{1 - \theta} \]
  \[ \leqslant (\lambda^{(7 - \theta) / (1 + \theta)} + \lambda^{(11 - \theta)
     / (1 + \theta)}) C_{\delta} \| \mathbb{X}_{\varepsilon} \|^8 + \delta
     \Upsilon_{\varepsilon} . \]
  For the paraproducts we have for $\theta = \frac{1 / 2 - 4 \kappa}{1 - 2
  \kappa}$
  \[ \lambda | \langle \rho^4 \phi_{\varepsilon}^2, X_{\varepsilon} \Join
     Y_{\varepsilon} \rangle_{\varepsilon} | \lesssim \lambda \| \rho^{4 - 2
     \sigma} \phi_{\varepsilon}^2 \|_{B^{1 / 2 + \kappa, \varepsilon}_{1, 1}}
     \| \rho^{\sigma} X_{\varepsilon} \|_{\mathscr{C} \hspace{.1em}^{- 1 / 2 -
     \kappa, \varepsilon}} \| \rho^{\sigma} Y_{\varepsilon} \|_{L^{\infty,
     \varepsilon}} \]
  \[ \lesssim \lambda^2 \| \mathbb{X}_{\varepsilon} \|^3 \| \rho
     \phi_{\varepsilon} \|^{1 + \theta}_{L^{4, \varepsilon}} \| \rho^2
     \phi_{\varepsilon} \|^{1 - \theta}_{H^{1 - 2 \kappa, \varepsilon}}
     \leqslant \lambda^{(7 - \theta) / (1 + \theta)} C_{\delta} \|
     \mathbb{X}_{\varepsilon} \|^8 + \delta \Upsilon_{\varepsilon} . \]
  With the improved bound for $Y$, {\eqref{eq:Y3}}, {\eqref{eq:Y2}}, {\eqref{eq:Y11}}
  can be updated as follows
  \[ | \langle \rho^4 \phi_{\varepsilon}, \lambda Y_{\varepsilon}^3
     \rangle_{\varepsilon} | \lesssim \lambda \| \rho \phi_{\varepsilon}
     \|_{L^{4, \varepsilon}} \| \rho^{\sigma} Y_{\varepsilon} \|_{C_T
     L^{\infty, \varepsilon}}^3 \lesssim \lambda^4 \| \rho \phi_{\varepsilon}
     \|_{L^{4, \varepsilon}} \| \mathbb{X}_{\varepsilon} \|^6 \leqslant \delta
     \lambda \| \rho \phi_{\varepsilon} \|_{L^{4, \varepsilon}}^4 + C_{\delta}
     \lambda^5 \| \mathbb{X}_{\varepsilon} \|^8, \]
  \[ | \langle \rho^4 \phi_{\varepsilon}, 3 \lambda Y_{\varepsilon}^2
     \phi_{\varepsilon} \rangle_{\varepsilon} | \lesssim \lambda \| \rho
     \phi_{\varepsilon} \|_{L^{4, \varepsilon}}^2 \| \rho^{\sigma}
     Y_{\varepsilon} \|_{C_T L^{\infty, \varepsilon}}^2 \lesssim \lambda^3 \|
     \rho \phi_{\varepsilon} \|_{L^{4, \varepsilon}}^2 \|
     \mathbb{X}_{\varepsilon} \|^4 \leqslant \delta \lambda \| \rho
     \phi_{\varepsilon} \|_{L^{4, \varepsilon}}^4 + C_{\delta} \lambda^5 \|
     \mathbb{X}_{\varepsilon} \|^8, \]
  \[ | \langle \rho^4 \phi_{\varepsilon}, 3 \lambda Y_{\varepsilon}
     \phi_{\varepsilon}^2 \rangle_{\varepsilon} | \lesssim \lambda \| \rho
     \phi_{\varepsilon} \|_{L^{4, \varepsilon}}^3 \| \rho^{\sigma}
     Y_{\varepsilon} \|_{C_T L^{\infty, \varepsilon}} \lesssim \lambda^2 \|
     \rho \phi_{\varepsilon} \|_{L^{4, \varepsilon}}^3 \|
     \mathbb{X}_{\varepsilon} \|^2 \leqslant \delta \lambda \| \rho
     \phi_{\varepsilon} \|_{L^{4, \varepsilon}}^4 + C_{\delta} \lambda^5 \|
     \mathbb{X}_{\varepsilon} \|^8 . \]
  Now, let us update the bound~{\eqref{eq:U0}} as
    \[ \lambda \left| \langle \rho^4 \phi_{\varepsilon}, - 3
     (\UU^{\varepsilon}_{\leqslant} \llbracket X^2 \rrbracket) \succ
     Y_{\varepsilon} \rangle_{\varepsilon} \right| \leqslant (\lambda^4 +
     \lambda^5) C_{\delta} \|\mathbb{X}_{\varepsilon} \|^{8 + \vartheta} +
     \delta \| \rho^2 \phi_{\varepsilon} \|_{H^{1 - 2 \kappa, \varepsilon}}^2
     . \]
  Next, we shall improve the bound~{\eqref{eq:Z0}}. Here we need to use a different modification
  for each term appearing in $\langle \rho^4 \phi_{\varepsilon}, \lambda^2
  Z_{\varepsilon} \rangle_{\varepsilon}$ as defined in~{\eqref{eq:def-Z}}. For
  $\theta = \frac{1 / 2 - 4 \kappa}{1 - 2 \kappa}$ we bound
  \[ | \langle \rho^4 \phi_{\varepsilon}, \lambda^2
     X_{\varepsilon}^{\!\resizebox{!}{.8em}{
\begin{tikzpicture}
\pgfpathmoveto{\pgfqpoint{0cm}{-0.035cm}}
\pgfpathlineto{\pgfqpoint{1.976cm}{-0.035cm}}
\pgfpathlineto{\pgfqpoint{1.976cm}{1.94cm}}
\pgfpathlineto{\pgfqpoint{0cm}{1.94cm}}
\pgfpathclose
\pgfusepath{clip}
\begin{pgfscope}
\begin{pgfscope}
\pgfpathmoveto{\pgfqpoint{0cm}{-0.035cm}}
\pgfpathlineto{\pgfqpoint{1.976cm}{-0.035cm}}
\pgfpathlineto{\pgfqpoint{1.976cm}{1.94cm}}
\pgfpathlineto{\pgfqpoint{0cm}{1.94cm}}
\pgfpathclose
\pgfusepath{clip}
\begin{pgfscope}
\begin{pgfscope}
\pgfsetdash{}{0cm}
\pgfsetlinewidth{0.818mm}
\pgfsetroundcap
\pgfsetroundjoin
\pgfsetmiterlimit{7.0}
\definecolor{eps2pgf_color}{gray}{0}\pgfsetstrokecolor{eps2pgf_color}\pgfsetfillcolor{eps2pgf_color}
\pgfpathmoveto{\pgfqpoint{0.117cm}{1.815cm}}
\pgfpathlineto{\pgfqpoint{0.682cm}{1.065cm}}
\pgfpathlineto{\pgfqpoint{1.246cm}{1.815cm}}
\pgfusepath{stroke}
\end{pgfscope}
\definecolor{eps2pgf_color}{gray}{0}\pgfsetstrokecolor{eps2pgf_color}\pgfsetfillcolor{eps2pgf_color}
\pgfpathmoveto{\pgfqpoint{0.273cm}{1.789cm}}
\pgfpathcurveto{\pgfqpoint{0.273cm}{1.825cm}}{\pgfqpoint{0.259cm}{1.86cm}}{\pgfqpoint{0.233cm}{1.886cm}}
\pgfpathcurveto{\pgfqpoint{0.207cm}{1.912cm}}{\pgfqpoint{0.173cm}{1.926cm}}{\pgfqpoint{0.137cm}{1.926cm}}
\pgfpathcurveto{\pgfqpoint{0.1cm}{1.926cm}}{\pgfqpoint{0.066cm}{1.912cm}}{\pgfqpoint{0.04cm}{1.886cm}}
\pgfpathcurveto{\pgfqpoint{0.014cm}{1.86cm}}{\pgfqpoint{0cm}{1.825cm}}{\pgfqpoint{0cm}{1.789cm}}
\pgfpathcurveto{\pgfqpoint{0cm}{1.753cm}}{\pgfqpoint{0.014cm}{1.718cm}}{\pgfqpoint{0.04cm}{1.692cm}}
\pgfpathcurveto{\pgfqpoint{0.066cm}{1.667cm}}{\pgfqpoint{0.1cm}{1.652cm}}{\pgfqpoint{0.137cm}{1.652cm}}
\pgfpathcurveto{\pgfqpoint{0.173cm}{1.652cm}}{\pgfqpoint{0.207cm}{1.667cm}}{\pgfqpoint{0.233cm}{1.692cm}}
\pgfpathcurveto{\pgfqpoint{0.259cm}{1.718cm}}{\pgfqpoint{0.273cm}{1.753cm}}{\pgfqpoint{0.273cm}{1.789cm}}
\pgfusepath{fill}
\begin{pgfscope}
\pgfsetdash{}{0cm}
\pgfsetlinewidth{0.818mm}
\pgfsetmiterlimit{7.0}
\pgfpathmoveto{\pgfqpoint{0.682cm}{1.065cm}}
\pgfpathlineto{\pgfqpoint{0.679cm}{1.812cm}}
\pgfusepath{stroke}
\end{pgfscope}
\pgfpathmoveto{\pgfqpoint{0.815cm}{1.793cm}}
\pgfpathcurveto{\pgfqpoint{0.815cm}{1.829cm}}{\pgfqpoint{0.801cm}{1.864cm}}{\pgfqpoint{0.775cm}{1.89cm}}
\pgfpathcurveto{\pgfqpoint{0.75cm}{1.915cm}}{\pgfqpoint{0.715cm}{1.93cm}}{\pgfqpoint{0.679cm}{1.93cm}}
\pgfpathcurveto{\pgfqpoint{0.643cm}{1.93cm}}{\pgfqpoint{0.608cm}{1.915cm}}{\pgfqpoint{0.582cm}{1.89cm}}
\pgfpathcurveto{\pgfqpoint{0.557cm}{1.864cm}}{\pgfqpoint{0.542cm}{1.829cm}}{\pgfqpoint{0.542cm}{1.793cm}}
\pgfpathcurveto{\pgfqpoint{0.542cm}{1.756cm}}{\pgfqpoint{0.557cm}{1.722cm}}{\pgfqpoint{0.582cm}{1.696cm}}
\pgfpathcurveto{\pgfqpoint{0.608cm}{1.67cm}}{\pgfqpoint{0.643cm}{1.656cm}}{\pgfqpoint{0.679cm}{1.656cm}}
\pgfpathcurveto{\pgfqpoint{0.715cm}{1.656cm}}{\pgfqpoint{0.75cm}{1.67cm}}{\pgfqpoint{0.775cm}{1.696cm}}
\pgfpathcurveto{\pgfqpoint{0.801cm}{1.722cm}}{\pgfqpoint{0.815cm}{1.756cm}}{\pgfqpoint{0.815cm}{1.793cm}}
\pgfusepath{fill}
\pgfpathmoveto{\pgfqpoint{1.345cm}{1.765cm}}
\pgfpathcurveto{\pgfqpoint{1.345cm}{1.801cm}}{\pgfqpoint{1.331cm}{1.836cm}}{\pgfqpoint{1.305cm}{1.862cm}}
\pgfpathcurveto{\pgfqpoint{1.28cm}{1.887cm}}{\pgfqpoint{1.245cm}{1.902cm}}{\pgfqpoint{1.209cm}{1.902cm}}
\pgfpathcurveto{\pgfqpoint{1.172cm}{1.902cm}}{\pgfqpoint{1.138cm}{1.887cm}}{\pgfqpoint{1.112cm}{1.862cm}}
\pgfpathcurveto{\pgfqpoint{1.087cm}{1.836cm}}{\pgfqpoint{1.072cm}{1.801cm}}{\pgfqpoint{1.072cm}{1.765cm}}
\pgfpathcurveto{\pgfqpoint{1.072cm}{1.728cm}}{\pgfqpoint{1.087cm}{1.694cm}}{\pgfqpoint{1.112cm}{1.668cm}}
\pgfpathcurveto{\pgfqpoint{1.138cm}{1.642cm}}{\pgfqpoint{1.172cm}{1.628cm}}{\pgfqpoint{1.209cm}{1.628cm}}
\pgfpathcurveto{\pgfqpoint{1.245cm}{1.628cm}}{\pgfqpoint{1.28cm}{1.642cm}}{\pgfqpoint{1.305cm}{1.668cm}}
\pgfpathcurveto{\pgfqpoint{1.331cm}{1.694cm}}{\pgfqpoint{1.345cm}{1.728cm}}{\pgfqpoint{1.345cm}{1.765cm}}
\pgfusepath{fill}
\begin{pgfscope}
\pgfsetdash{}{0cm}
\pgfsetlinewidth{0.818mm}
\pgfsetroundcap
\pgfsetroundjoin
\pgfsetmiterlimit{7.0}
\pgfpathmoveto{\pgfqpoint{0.682cm}{1.065cm}}
\pgfpathlineto{\pgfqpoint{1.246cm}{0.315cm}}
\pgfpathlineto{\pgfqpoint{1.811cm}{1.065cm}}
\pgfusepath{stroke}
\end{pgfscope}
\pgfpathmoveto{\pgfqpoint{1.948cm}{1.065cm}}
\pgfpathcurveto{\pgfqpoint{1.948cm}{1.101cm}}{\pgfqpoint{1.933cm}{1.136cm}}{\pgfqpoint{1.907cm}{1.162cm}}
\pgfpathcurveto{\pgfqpoint{1.882cm}{1.187cm}}{\pgfqpoint{1.847cm}{1.202cm}}{\pgfqpoint{1.811cm}{1.202cm}}
\pgfpathcurveto{\pgfqpoint{1.775cm}{1.202cm}}{\pgfqpoint{1.74cm}{1.187cm}}{\pgfqpoint{1.714cm}{1.162cm}}
\pgfpathcurveto{\pgfqpoint{1.689cm}{1.136cm}}{\pgfqpoint{1.674cm}{1.101cm}}{\pgfqpoint{1.674cm}{1.065cm}}
\pgfpathcurveto{\pgfqpoint{1.674cm}{1.029cm}}{\pgfqpoint{1.689cm}{0.994cm}}{\pgfqpoint{1.714cm}{0.968cm}}
\pgfpathcurveto{\pgfqpoint{1.74cm}{0.942cm}}{\pgfqpoint{1.775cm}{0.928cm}}{\pgfqpoint{1.811cm}{0.928cm}}
\pgfpathcurveto{\pgfqpoint{1.847cm}{0.928cm}}{\pgfqpoint{1.882cm}{0.942cm}}{\pgfqpoint{1.907cm}{0.968cm}}
\pgfpathcurveto{\pgfqpoint{1.933cm}{0.994cm}}{\pgfqpoint{1.948cm}{1.029cm}}{\pgfqpoint{1.948cm}{1.065cm}}
\pgfusepath{fill}
\begin{pgfscope}
\pgfsetdash{}{0cm}
\pgfsetlinewidth{0.818mm}
\pgfsetmiterlimit{7.0}
\pgfpathmoveto{\pgfqpoint{1.246cm}{0.315cm}}
\pgfpathlineto{\pgfqpoint{1.244cm}{1.061cm}}
\pgfusepath{stroke}
\end{pgfscope}
\pgfpathmoveto{\pgfqpoint{1.38cm}{1.065cm}}
\pgfpathcurveto{\pgfqpoint{1.38cm}{1.101cm}}{\pgfqpoint{1.366cm}{1.136cm}}{\pgfqpoint{1.34cm}{1.162cm}}
\pgfpathcurveto{\pgfqpoint{1.315cm}{1.187cm}}{\pgfqpoint{1.28cm}{1.202cm}}{\pgfqpoint{1.244cm}{1.202cm}}
\pgfpathcurveto{\pgfqpoint{1.207cm}{1.202cm}}{\pgfqpoint{1.173cm}{1.187cm}}{\pgfqpoint{1.147cm}{1.162cm}}
\pgfpathcurveto{\pgfqpoint{1.121cm}{1.136cm}}{\pgfqpoint{1.107cm}{1.101cm}}{\pgfqpoint{1.107cm}{1.065cm}}
\pgfpathcurveto{\pgfqpoint{1.107cm}{1.029cm}}{\pgfqpoint{1.121cm}{0.994cm}}{\pgfqpoint{1.147cm}{0.968cm}}
\pgfpathcurveto{\pgfqpoint{1.173cm}{0.942cm}}{\pgfqpoint{1.207cm}{0.928cm}}{\pgfqpoint{1.244cm}{0.928cm}}
\pgfpathcurveto{\pgfqpoint{1.28cm}{0.928cm}}{\pgfqpoint{1.315cm}{0.942cm}}{\pgfqpoint{1.34cm}{0.968cm}}
\pgfpathcurveto{\pgfqpoint{1.366cm}{0.994cm}}{\pgfqpoint{1.38cm}{1.029cm}}{\pgfqpoint{1.38cm}{1.065cm}}
\pgfusepath{fill}
\begin{pgfscope}
\pgfsetdash{}{0cm}
\pgfsetlinewidth{0.818mm}
\pgfsetmiterlimit{4.0}
\pgfpathmoveto{\pgfqpoint{1.383cm}{0.178cm}}
\pgfpathcurveto{\pgfqpoint{1.383cm}{0.214cm}}{\pgfqpoint{1.369cm}{0.249cm}}{\pgfqpoint{1.343cm}{0.275cm}}
\pgfpathcurveto{\pgfqpoint{1.317cm}{0.3cm}}{\pgfqpoint{1.283cm}{0.315cm}}{\pgfqpoint{1.246cm}{0.315cm}}
\pgfpathcurveto{\pgfqpoint{1.21cm}{0.315cm}}{\pgfqpoint{1.175cm}{0.3cm}}{\pgfqpoint{1.15cm}{0.275cm}}
\pgfpathcurveto{\pgfqpoint{1.124cm}{0.249cm}}{\pgfqpoint{1.11cm}{0.214cm}}{\pgfqpoint{1.11cm}{0.178cm}}
\pgfpathcurveto{\pgfqpoint{1.11cm}{0.141cm}}{\pgfqpoint{1.124cm}{0.107cm}}{\pgfqpoint{1.15cm}{0.081cm}}
\pgfpathcurveto{\pgfqpoint{1.175cm}{0.055cm}}{\pgfqpoint{1.21cm}{0.041cm}}{\pgfqpoint{1.246cm}{0.041cm}}
\pgfpathcurveto{\pgfqpoint{1.283cm}{0.041cm}}{\pgfqpoint{1.317cm}{0.055cm}}{\pgfqpoint{1.343cm}{0.081cm}}
\pgfpathcurveto{\pgfqpoint{1.369cm}{0.107cm}}{\pgfqpoint{1.383cm}{0.141cm}}{\pgfqpoint{1.383cm}{0.178cm}}
\pgfusepath{stroke}
\end{pgfscope}
\end{pgfscope}
\end{pgfscope}
\end{pgfscope}
\end{tikzpicture}}} \rangle_{\varepsilon} | \lesssim \lambda^2
     \| \rho^{4 - \sigma} \phi_{\varepsilon} \|_{B^{1 / 2 + \kappa,
     \varepsilon}_{1, 1}} \| \rho^{\sigma} X_{\varepsilon}^{\!\resizebox{!}{.8em}{
\begin{tikzpicture}
\pgfpathmoveto{\pgfqpoint{0cm}{-0.035cm}}
\pgfpathlineto{\pgfqpoint{1.976cm}{-0.035cm}}
\pgfpathlineto{\pgfqpoint{1.976cm}{1.94cm}}
\pgfpathlineto{\pgfqpoint{0cm}{1.94cm}}
\pgfpathclose
\pgfusepath{clip}
\begin{pgfscope}
\begin{pgfscope}
\pgfpathmoveto{\pgfqpoint{0cm}{-0.035cm}}
\pgfpathlineto{\pgfqpoint{1.976cm}{-0.035cm}}
\pgfpathlineto{\pgfqpoint{1.976cm}{1.94cm}}
\pgfpathlineto{\pgfqpoint{0cm}{1.94cm}}
\pgfpathclose
\pgfusepath{clip}
\begin{pgfscope}
\begin{pgfscope}
\pgfsetdash{}{0cm}
\pgfsetlinewidth{0.818mm}
\pgfsetroundcap
\pgfsetroundjoin
\pgfsetmiterlimit{7.0}
\definecolor{eps2pgf_color}{gray}{0}\pgfsetstrokecolor{eps2pgf_color}\pgfsetfillcolor{eps2pgf_color}
\pgfpathmoveto{\pgfqpoint{0.117cm}{1.815cm}}
\pgfpathlineto{\pgfqpoint{0.682cm}{1.065cm}}
\pgfpathlineto{\pgfqpoint{1.246cm}{1.815cm}}
\pgfusepath{stroke}
\end{pgfscope}
\definecolor{eps2pgf_color}{gray}{0}\pgfsetstrokecolor{eps2pgf_color}\pgfsetfillcolor{eps2pgf_color}
\pgfpathmoveto{\pgfqpoint{0.273cm}{1.789cm}}
\pgfpathcurveto{\pgfqpoint{0.273cm}{1.825cm}}{\pgfqpoint{0.259cm}{1.86cm}}{\pgfqpoint{0.233cm}{1.886cm}}
\pgfpathcurveto{\pgfqpoint{0.207cm}{1.912cm}}{\pgfqpoint{0.173cm}{1.926cm}}{\pgfqpoint{0.137cm}{1.926cm}}
\pgfpathcurveto{\pgfqpoint{0.1cm}{1.926cm}}{\pgfqpoint{0.066cm}{1.912cm}}{\pgfqpoint{0.04cm}{1.886cm}}
\pgfpathcurveto{\pgfqpoint{0.014cm}{1.86cm}}{\pgfqpoint{0cm}{1.825cm}}{\pgfqpoint{0cm}{1.789cm}}
\pgfpathcurveto{\pgfqpoint{0cm}{1.753cm}}{\pgfqpoint{0.014cm}{1.718cm}}{\pgfqpoint{0.04cm}{1.692cm}}
\pgfpathcurveto{\pgfqpoint{0.066cm}{1.667cm}}{\pgfqpoint{0.1cm}{1.652cm}}{\pgfqpoint{0.137cm}{1.652cm}}
\pgfpathcurveto{\pgfqpoint{0.173cm}{1.652cm}}{\pgfqpoint{0.207cm}{1.667cm}}{\pgfqpoint{0.233cm}{1.692cm}}
\pgfpathcurveto{\pgfqpoint{0.259cm}{1.718cm}}{\pgfqpoint{0.273cm}{1.753cm}}{\pgfqpoint{0.273cm}{1.789cm}}
\pgfusepath{fill}
\begin{pgfscope}
\pgfsetdash{}{0cm}
\pgfsetlinewidth{0.818mm}
\pgfsetmiterlimit{7.0}
\pgfpathmoveto{\pgfqpoint{0.682cm}{1.065cm}}
\pgfpathlineto{\pgfqpoint{0.679cm}{1.812cm}}
\pgfusepath{stroke}
\end{pgfscope}
\pgfpathmoveto{\pgfqpoint{0.815cm}{1.793cm}}
\pgfpathcurveto{\pgfqpoint{0.815cm}{1.829cm}}{\pgfqpoint{0.801cm}{1.864cm}}{\pgfqpoint{0.775cm}{1.89cm}}
\pgfpathcurveto{\pgfqpoint{0.75cm}{1.915cm}}{\pgfqpoint{0.715cm}{1.93cm}}{\pgfqpoint{0.679cm}{1.93cm}}
\pgfpathcurveto{\pgfqpoint{0.643cm}{1.93cm}}{\pgfqpoint{0.608cm}{1.915cm}}{\pgfqpoint{0.582cm}{1.89cm}}
\pgfpathcurveto{\pgfqpoint{0.557cm}{1.864cm}}{\pgfqpoint{0.542cm}{1.829cm}}{\pgfqpoint{0.542cm}{1.793cm}}
\pgfpathcurveto{\pgfqpoint{0.542cm}{1.756cm}}{\pgfqpoint{0.557cm}{1.722cm}}{\pgfqpoint{0.582cm}{1.696cm}}
\pgfpathcurveto{\pgfqpoint{0.608cm}{1.67cm}}{\pgfqpoint{0.643cm}{1.656cm}}{\pgfqpoint{0.679cm}{1.656cm}}
\pgfpathcurveto{\pgfqpoint{0.715cm}{1.656cm}}{\pgfqpoint{0.75cm}{1.67cm}}{\pgfqpoint{0.775cm}{1.696cm}}
\pgfpathcurveto{\pgfqpoint{0.801cm}{1.722cm}}{\pgfqpoint{0.815cm}{1.756cm}}{\pgfqpoint{0.815cm}{1.793cm}}
\pgfusepath{fill}
\pgfpathmoveto{\pgfqpoint{1.345cm}{1.765cm}}
\pgfpathcurveto{\pgfqpoint{1.345cm}{1.801cm}}{\pgfqpoint{1.331cm}{1.836cm}}{\pgfqpoint{1.305cm}{1.862cm}}
\pgfpathcurveto{\pgfqpoint{1.28cm}{1.887cm}}{\pgfqpoint{1.245cm}{1.902cm}}{\pgfqpoint{1.209cm}{1.902cm}}
\pgfpathcurveto{\pgfqpoint{1.172cm}{1.902cm}}{\pgfqpoint{1.138cm}{1.887cm}}{\pgfqpoint{1.112cm}{1.862cm}}
\pgfpathcurveto{\pgfqpoint{1.087cm}{1.836cm}}{\pgfqpoint{1.072cm}{1.801cm}}{\pgfqpoint{1.072cm}{1.765cm}}
\pgfpathcurveto{\pgfqpoint{1.072cm}{1.728cm}}{\pgfqpoint{1.087cm}{1.694cm}}{\pgfqpoint{1.112cm}{1.668cm}}
\pgfpathcurveto{\pgfqpoint{1.138cm}{1.642cm}}{\pgfqpoint{1.172cm}{1.628cm}}{\pgfqpoint{1.209cm}{1.628cm}}
\pgfpathcurveto{\pgfqpoint{1.245cm}{1.628cm}}{\pgfqpoint{1.28cm}{1.642cm}}{\pgfqpoint{1.305cm}{1.668cm}}
\pgfpathcurveto{\pgfqpoint{1.331cm}{1.694cm}}{\pgfqpoint{1.345cm}{1.728cm}}{\pgfqpoint{1.345cm}{1.765cm}}
\pgfusepath{fill}
\begin{pgfscope}
\pgfsetdash{}{0cm}
\pgfsetlinewidth{0.818mm}
\pgfsetroundcap
\pgfsetroundjoin
\pgfsetmiterlimit{7.0}
\pgfpathmoveto{\pgfqpoint{0.682cm}{1.065cm}}
\pgfpathlineto{\pgfqpoint{1.246cm}{0.315cm}}
\pgfpathlineto{\pgfqpoint{1.811cm}{1.065cm}}
\pgfusepath{stroke}
\end{pgfscope}
\pgfpathmoveto{\pgfqpoint{1.948cm}{1.065cm}}
\pgfpathcurveto{\pgfqpoint{1.948cm}{1.101cm}}{\pgfqpoint{1.933cm}{1.136cm}}{\pgfqpoint{1.907cm}{1.162cm}}
\pgfpathcurveto{\pgfqpoint{1.882cm}{1.187cm}}{\pgfqpoint{1.847cm}{1.202cm}}{\pgfqpoint{1.811cm}{1.202cm}}
\pgfpathcurveto{\pgfqpoint{1.775cm}{1.202cm}}{\pgfqpoint{1.74cm}{1.187cm}}{\pgfqpoint{1.714cm}{1.162cm}}
\pgfpathcurveto{\pgfqpoint{1.689cm}{1.136cm}}{\pgfqpoint{1.674cm}{1.101cm}}{\pgfqpoint{1.674cm}{1.065cm}}
\pgfpathcurveto{\pgfqpoint{1.674cm}{1.029cm}}{\pgfqpoint{1.689cm}{0.994cm}}{\pgfqpoint{1.714cm}{0.968cm}}
\pgfpathcurveto{\pgfqpoint{1.74cm}{0.942cm}}{\pgfqpoint{1.775cm}{0.928cm}}{\pgfqpoint{1.811cm}{0.928cm}}
\pgfpathcurveto{\pgfqpoint{1.847cm}{0.928cm}}{\pgfqpoint{1.882cm}{0.942cm}}{\pgfqpoint{1.907cm}{0.968cm}}
\pgfpathcurveto{\pgfqpoint{1.933cm}{0.994cm}}{\pgfqpoint{1.948cm}{1.029cm}}{\pgfqpoint{1.948cm}{1.065cm}}
\pgfusepath{fill}
\begin{pgfscope}
\pgfsetdash{}{0cm}
\pgfsetlinewidth{0.818mm}
\pgfsetmiterlimit{7.0}
\pgfpathmoveto{\pgfqpoint{1.246cm}{0.315cm}}
\pgfpathlineto{\pgfqpoint{1.244cm}{1.061cm}}
\pgfusepath{stroke}
\end{pgfscope}
\pgfpathmoveto{\pgfqpoint{1.38cm}{1.065cm}}
\pgfpathcurveto{\pgfqpoint{1.38cm}{1.101cm}}{\pgfqpoint{1.366cm}{1.136cm}}{\pgfqpoint{1.34cm}{1.162cm}}
\pgfpathcurveto{\pgfqpoint{1.315cm}{1.187cm}}{\pgfqpoint{1.28cm}{1.202cm}}{\pgfqpoint{1.244cm}{1.202cm}}
\pgfpathcurveto{\pgfqpoint{1.207cm}{1.202cm}}{\pgfqpoint{1.173cm}{1.187cm}}{\pgfqpoint{1.147cm}{1.162cm}}
\pgfpathcurveto{\pgfqpoint{1.121cm}{1.136cm}}{\pgfqpoint{1.107cm}{1.101cm}}{\pgfqpoint{1.107cm}{1.065cm}}
\pgfpathcurveto{\pgfqpoint{1.107cm}{1.029cm}}{\pgfqpoint{1.121cm}{0.994cm}}{\pgfqpoint{1.147cm}{0.968cm}}
\pgfpathcurveto{\pgfqpoint{1.173cm}{0.942cm}}{\pgfqpoint{1.207cm}{0.928cm}}{\pgfqpoint{1.244cm}{0.928cm}}
\pgfpathcurveto{\pgfqpoint{1.28cm}{0.928cm}}{\pgfqpoint{1.315cm}{0.942cm}}{\pgfqpoint{1.34cm}{0.968cm}}
\pgfpathcurveto{\pgfqpoint{1.366cm}{0.994cm}}{\pgfqpoint{1.38cm}{1.029cm}}{\pgfqpoint{1.38cm}{1.065cm}}
\pgfusepath{fill}
\begin{pgfscope}
\pgfsetdash{}{0cm}
\pgfsetlinewidth{0.818mm}
\pgfsetmiterlimit{4.0}
\pgfpathmoveto{\pgfqpoint{1.383cm}{0.178cm}}
\pgfpathcurveto{\pgfqpoint{1.383cm}{0.214cm}}{\pgfqpoint{1.369cm}{0.249cm}}{\pgfqpoint{1.343cm}{0.275cm}}
\pgfpathcurveto{\pgfqpoint{1.317cm}{0.3cm}}{\pgfqpoint{1.283cm}{0.315cm}}{\pgfqpoint{1.246cm}{0.315cm}}
\pgfpathcurveto{\pgfqpoint{1.21cm}{0.315cm}}{\pgfqpoint{1.175cm}{0.3cm}}{\pgfqpoint{1.15cm}{0.275cm}}
\pgfpathcurveto{\pgfqpoint{1.124cm}{0.249cm}}{\pgfqpoint{1.11cm}{0.214cm}}{\pgfqpoint{1.11cm}{0.178cm}}
\pgfpathcurveto{\pgfqpoint{1.11cm}{0.141cm}}{\pgfqpoint{1.124cm}{0.107cm}}{\pgfqpoint{1.15cm}{0.081cm}}
\pgfpathcurveto{\pgfqpoint{1.175cm}{0.055cm}}{\pgfqpoint{1.21cm}{0.041cm}}{\pgfqpoint{1.246cm}{0.041cm}}
\pgfpathcurveto{\pgfqpoint{1.283cm}{0.041cm}}{\pgfqpoint{1.317cm}{0.055cm}}{\pgfqpoint{1.343cm}{0.081cm}}
\pgfpathcurveto{\pgfqpoint{1.369cm}{0.107cm}}{\pgfqpoint{1.383cm}{0.141cm}}{\pgfqpoint{1.383cm}{0.178cm}}
\pgfusepath{stroke}
\end{pgfscope}
\end{pgfscope}
\end{pgfscope}
\end{pgfscope}
\end{tikzpicture}}}
     \|_{C_T \mathscr{C} \hspace{.1em}^{- 1 / 2 - \kappa, \varepsilon}} \]
  \[ \lesssim \lambda^2 \| \rho \phi_{\varepsilon} \|_{L^{4,
     \varepsilon}}^{\theta} \| \rho^2 \phi_{\varepsilon} \|_{H^{1 - 2 \kappa,
     \varepsilon}}^{1 - \theta} \| \mathbb{X}_{\varepsilon} \|^5 \leqslant
     \lambda^{(8 - \theta) / (2 + \theta)} C_{\delta} \|
     \mathbb{X}_{\varepsilon} \|^8 + \delta \Upsilon_{\varepsilon} \]
  \[ \leqslant (\lambda^3 + \lambda^4) C_{\delta} \| \mathbb{X}_{\varepsilon}
     \|^8 + \delta \Upsilon_{\varepsilon}. \]
  Next, we have
  \[ \lambda^2 | \langle \rho^4 \phi_{\varepsilon},
     \tilde{X}_{\varepsilon}^{\!\resizebox{!}{.8em}{
\begin{tikzpicture}
\pgfpathmoveto{\pgfqpoint{0cm}{-0.035cm}}
\pgfpathlineto{\pgfqpoint{1.976cm}{-0.035cm}}
\pgfpathlineto{\pgfqpoint{1.976cm}{1.94cm}}
\pgfpathlineto{\pgfqpoint{0cm}{1.94cm}}
\pgfpathclose
\pgfusepath{clip}
\begin{pgfscope}
\begin{pgfscope}
\pgfpathmoveto{\pgfqpoint{0cm}{-0.035cm}}
\pgfpathlineto{\pgfqpoint{1.976cm}{-0.035cm}}
\pgfpathlineto{\pgfqpoint{1.976cm}{1.94cm}}
\pgfpathlineto{\pgfqpoint{0cm}{1.94cm}}
\pgfpathclose
\pgfusepath{clip}
\begin{pgfscope}
\begin{pgfscope}
\pgfsetdash{}{0cm}
\pgfsetlinewidth{0.818mm}
\pgfsetroundcap
\pgfsetroundjoin
\pgfsetmiterlimit{7.0}
\definecolor{eps2pgf_color}{gray}{0}\pgfsetstrokecolor{eps2pgf_color}\pgfsetfillcolor{eps2pgf_color}
\pgfpathmoveto{\pgfqpoint{0.117cm}{1.815cm}}
\pgfpathlineto{\pgfqpoint{0.682cm}{1.065cm}}
\pgfpathlineto{\pgfqpoint{1.246cm}{1.815cm}}
\pgfusepath{stroke}
\end{pgfscope}
\definecolor{eps2pgf_color}{gray}{0}\pgfsetstrokecolor{eps2pgf_color}\pgfsetfillcolor{eps2pgf_color}
\pgfpathmoveto{\pgfqpoint{0.273cm}{1.789cm}}
\pgfpathcurveto{\pgfqpoint{0.273cm}{1.825cm}}{\pgfqpoint{0.259cm}{1.86cm}}{\pgfqpoint{0.233cm}{1.886cm}}
\pgfpathcurveto{\pgfqpoint{0.207cm}{1.912cm}}{\pgfqpoint{0.173cm}{1.926cm}}{\pgfqpoint{0.137cm}{1.926cm}}
\pgfpathcurveto{\pgfqpoint{0.1cm}{1.926cm}}{\pgfqpoint{0.066cm}{1.912cm}}{\pgfqpoint{0.04cm}{1.886cm}}
\pgfpathcurveto{\pgfqpoint{0.014cm}{1.86cm}}{\pgfqpoint{0cm}{1.825cm}}{\pgfqpoint{0cm}{1.789cm}}
\pgfpathcurveto{\pgfqpoint{0cm}{1.753cm}}{\pgfqpoint{0.014cm}{1.718cm}}{\pgfqpoint{0.04cm}{1.692cm}}
\pgfpathcurveto{\pgfqpoint{0.066cm}{1.667cm}}{\pgfqpoint{0.1cm}{1.652cm}}{\pgfqpoint{0.137cm}{1.652cm}}
\pgfpathcurveto{\pgfqpoint{0.173cm}{1.652cm}}{\pgfqpoint{0.207cm}{1.667cm}}{\pgfqpoint{0.233cm}{1.692cm}}
\pgfpathcurveto{\pgfqpoint{0.259cm}{1.718cm}}{\pgfqpoint{0.273cm}{1.753cm}}{\pgfqpoint{0.273cm}{1.789cm}}
\pgfusepath{fill}
\pgfpathmoveto{\pgfqpoint{1.345cm}{1.765cm}}
\pgfpathcurveto{\pgfqpoint{1.345cm}{1.801cm}}{\pgfqpoint{1.331cm}{1.836cm}}{\pgfqpoint{1.305cm}{1.862cm}}
\pgfpathcurveto{\pgfqpoint{1.28cm}{1.887cm}}{\pgfqpoint{1.245cm}{1.902cm}}{\pgfqpoint{1.209cm}{1.902cm}}
\pgfpathcurveto{\pgfqpoint{1.172cm}{1.902cm}}{\pgfqpoint{1.138cm}{1.887cm}}{\pgfqpoint{1.112cm}{1.862cm}}
\pgfpathcurveto{\pgfqpoint{1.087cm}{1.836cm}}{\pgfqpoint{1.072cm}{1.801cm}}{\pgfqpoint{1.072cm}{1.765cm}}
\pgfpathcurveto{\pgfqpoint{1.072cm}{1.728cm}}{\pgfqpoint{1.087cm}{1.694cm}}{\pgfqpoint{1.112cm}{1.668cm}}
\pgfpathcurveto{\pgfqpoint{1.138cm}{1.642cm}}{\pgfqpoint{1.172cm}{1.628cm}}{\pgfqpoint{1.209cm}{1.628cm}}
\pgfpathcurveto{\pgfqpoint{1.245cm}{1.628cm}}{\pgfqpoint{1.28cm}{1.642cm}}{\pgfqpoint{1.305cm}{1.668cm}}
\pgfpathcurveto{\pgfqpoint{1.331cm}{1.694cm}}{\pgfqpoint{1.345cm}{1.728cm}}{\pgfqpoint{1.345cm}{1.765cm}}
\pgfusepath{fill}
\begin{pgfscope}
\pgfsetdash{}{0cm}
\pgfsetlinewidth{0.818mm}
\pgfsetroundcap
\pgfsetroundjoin
\pgfsetmiterlimit{7.0}
\pgfpathmoveto{\pgfqpoint{0.682cm}{1.065cm}}
\pgfpathlineto{\pgfqpoint{1.246cm}{0.315cm}}
\pgfpathlineto{\pgfqpoint{1.811cm}{1.065cm}}
\pgfusepath{stroke}
\end{pgfscope}
\pgfpathmoveto{\pgfqpoint{1.948cm}{1.065cm}}
\pgfpathcurveto{\pgfqpoint{1.948cm}{1.101cm}}{\pgfqpoint{1.933cm}{1.136cm}}{\pgfqpoint{1.907cm}{1.162cm}}
\pgfpathcurveto{\pgfqpoint{1.882cm}{1.187cm}}{\pgfqpoint{1.847cm}{1.202cm}}{\pgfqpoint{1.811cm}{1.202cm}}
\pgfpathcurveto{\pgfqpoint{1.775cm}{1.202cm}}{\pgfqpoint{1.74cm}{1.187cm}}{\pgfqpoint{1.714cm}{1.162cm}}
\pgfpathcurveto{\pgfqpoint{1.689cm}{1.136cm}}{\pgfqpoint{1.674cm}{1.101cm}}{\pgfqpoint{1.674cm}{1.065cm}}
\pgfpathcurveto{\pgfqpoint{1.674cm}{1.029cm}}{\pgfqpoint{1.689cm}{0.994cm}}{\pgfqpoint{1.714cm}{0.968cm}}
\pgfpathcurveto{\pgfqpoint{1.74cm}{0.942cm}}{\pgfqpoint{1.775cm}{0.928cm}}{\pgfqpoint{1.811cm}{0.928cm}}
\pgfpathcurveto{\pgfqpoint{1.847cm}{0.928cm}}{\pgfqpoint{1.882cm}{0.942cm}}{\pgfqpoint{1.907cm}{0.968cm}}
\pgfpathcurveto{\pgfqpoint{1.933cm}{0.994cm}}{\pgfqpoint{1.948cm}{1.029cm}}{\pgfqpoint{1.948cm}{1.065cm}}
\pgfusepath{fill}
\begin{pgfscope}
\pgfsetdash{}{0cm}
\pgfsetlinewidth{0.818mm}
\pgfsetmiterlimit{7.0}
\pgfpathmoveto{\pgfqpoint{1.246cm}{0.315cm}}
\pgfpathlineto{\pgfqpoint{1.244cm}{1.061cm}}
\pgfusepath{stroke}
\end{pgfscope}
\pgfpathmoveto{\pgfqpoint{1.38cm}{1.065cm}}
\pgfpathcurveto{\pgfqpoint{1.38cm}{1.101cm}}{\pgfqpoint{1.366cm}{1.136cm}}{\pgfqpoint{1.34cm}{1.162cm}}
\pgfpathcurveto{\pgfqpoint{1.315cm}{1.187cm}}{\pgfqpoint{1.28cm}{1.202cm}}{\pgfqpoint{1.244cm}{1.202cm}}
\pgfpathcurveto{\pgfqpoint{1.207cm}{1.202cm}}{\pgfqpoint{1.173cm}{1.187cm}}{\pgfqpoint{1.147cm}{1.162cm}}
\pgfpathcurveto{\pgfqpoint{1.121cm}{1.136cm}}{\pgfqpoint{1.107cm}{1.101cm}}{\pgfqpoint{1.107cm}{1.065cm}}
\pgfpathcurveto{\pgfqpoint{1.107cm}{1.029cm}}{\pgfqpoint{1.121cm}{0.994cm}}{\pgfqpoint{1.147cm}{0.968cm}}
\pgfpathcurveto{\pgfqpoint{1.173cm}{0.942cm}}{\pgfqpoint{1.207cm}{0.928cm}}{\pgfqpoint{1.244cm}{0.928cm}}
\pgfpathcurveto{\pgfqpoint{1.28cm}{0.928cm}}{\pgfqpoint{1.315cm}{0.942cm}}{\pgfqpoint{1.34cm}{0.968cm}}
\pgfpathcurveto{\pgfqpoint{1.366cm}{0.994cm}}{\pgfqpoint{1.38cm}{1.029cm}}{\pgfqpoint{1.38cm}{1.065cm}}
\pgfusepath{fill}
\begin{pgfscope}
\pgfsetdash{}{0cm}
\pgfsetlinewidth{0.818mm}
\pgfsetmiterlimit{4.0}
\pgfpathmoveto{\pgfqpoint{1.383cm}{0.178cm}}
\pgfpathcurveto{\pgfqpoint{1.383cm}{0.214cm}}{\pgfqpoint{1.369cm}{0.249cm}}{\pgfqpoint{1.343cm}{0.275cm}}
\pgfpathcurveto{\pgfqpoint{1.317cm}{0.3cm}}{\pgfqpoint{1.283cm}{0.315cm}}{\pgfqpoint{1.246cm}{0.315cm}}
\pgfpathcurveto{\pgfqpoint{1.21cm}{0.315cm}}{\pgfqpoint{1.175cm}{0.3cm}}{\pgfqpoint{1.15cm}{0.275cm}}
\pgfpathcurveto{\pgfqpoint{1.124cm}{0.249cm}}{\pgfqpoint{1.11cm}{0.214cm}}{\pgfqpoint{1.11cm}{0.178cm}}
\pgfpathcurveto{\pgfqpoint{1.11cm}{0.141cm}}{\pgfqpoint{1.124cm}{0.107cm}}{\pgfqpoint{1.15cm}{0.081cm}}
\pgfpathcurveto{\pgfqpoint{1.175cm}{0.055cm}}{\pgfqpoint{1.21cm}{0.041cm}}{\pgfqpoint{1.246cm}{0.041cm}}
\pgfpathcurveto{\pgfqpoint{1.283cm}{0.041cm}}{\pgfqpoint{1.317cm}{0.055cm}}{\pgfqpoint{1.343cm}{0.081cm}}
\pgfpathcurveto{\pgfqpoint{1.369cm}{0.107cm}}{\pgfqpoint{1.383cm}{0.141cm}}{\pgfqpoint{1.383cm}{0.178cm}}
\pgfusepath{stroke}
\end{pgfscope}
\end{pgfscope}
\end{pgfscope}
\end{pgfscope}
\end{tikzpicture}}} Y \rangle_{\varepsilon} | \leqslant
     \lambda^2 | \langle \rho^4 \phi_{\varepsilon},
     \tilde{X}_{\varepsilon}^{\!\resizebox{!}{.8em}{
\begin{tikzpicture}
\pgfpathmoveto{\pgfqpoint{0cm}{-0.035cm}}
\pgfpathlineto{\pgfqpoint{1.976cm}{-0.035cm}}
\pgfpathlineto{\pgfqpoint{1.976cm}{1.94cm}}
\pgfpathlineto{\pgfqpoint{0cm}{1.94cm}}
\pgfpathclose
\pgfusepath{clip}
\begin{pgfscope}
\begin{pgfscope}
\pgfpathmoveto{\pgfqpoint{0cm}{-0.035cm}}
\pgfpathlineto{\pgfqpoint{1.976cm}{-0.035cm}}
\pgfpathlineto{\pgfqpoint{1.976cm}{1.94cm}}
\pgfpathlineto{\pgfqpoint{0cm}{1.94cm}}
\pgfpathclose
\pgfusepath{clip}
\begin{pgfscope}
\begin{pgfscope}
\pgfsetdash{}{0cm}
\pgfsetlinewidth{0.818mm}
\pgfsetroundcap
\pgfsetroundjoin
\pgfsetmiterlimit{7.0}
\definecolor{eps2pgf_color}{gray}{0}\pgfsetstrokecolor{eps2pgf_color}\pgfsetfillcolor{eps2pgf_color}
\pgfpathmoveto{\pgfqpoint{0.117cm}{1.815cm}}
\pgfpathlineto{\pgfqpoint{0.682cm}{1.065cm}}
\pgfpathlineto{\pgfqpoint{1.246cm}{1.815cm}}
\pgfusepath{stroke}
\end{pgfscope}
\definecolor{eps2pgf_color}{gray}{0}\pgfsetstrokecolor{eps2pgf_color}\pgfsetfillcolor{eps2pgf_color}
\pgfpathmoveto{\pgfqpoint{0.273cm}{1.789cm}}
\pgfpathcurveto{\pgfqpoint{0.273cm}{1.825cm}}{\pgfqpoint{0.259cm}{1.86cm}}{\pgfqpoint{0.233cm}{1.886cm}}
\pgfpathcurveto{\pgfqpoint{0.207cm}{1.912cm}}{\pgfqpoint{0.173cm}{1.926cm}}{\pgfqpoint{0.137cm}{1.926cm}}
\pgfpathcurveto{\pgfqpoint{0.1cm}{1.926cm}}{\pgfqpoint{0.066cm}{1.912cm}}{\pgfqpoint{0.04cm}{1.886cm}}
\pgfpathcurveto{\pgfqpoint{0.014cm}{1.86cm}}{\pgfqpoint{0cm}{1.825cm}}{\pgfqpoint{0cm}{1.789cm}}
\pgfpathcurveto{\pgfqpoint{0cm}{1.753cm}}{\pgfqpoint{0.014cm}{1.718cm}}{\pgfqpoint{0.04cm}{1.692cm}}
\pgfpathcurveto{\pgfqpoint{0.066cm}{1.667cm}}{\pgfqpoint{0.1cm}{1.652cm}}{\pgfqpoint{0.137cm}{1.652cm}}
\pgfpathcurveto{\pgfqpoint{0.173cm}{1.652cm}}{\pgfqpoint{0.207cm}{1.667cm}}{\pgfqpoint{0.233cm}{1.692cm}}
\pgfpathcurveto{\pgfqpoint{0.259cm}{1.718cm}}{\pgfqpoint{0.273cm}{1.753cm}}{\pgfqpoint{0.273cm}{1.789cm}}
\pgfusepath{fill}
\pgfpathmoveto{\pgfqpoint{1.345cm}{1.765cm}}
\pgfpathcurveto{\pgfqpoint{1.345cm}{1.801cm}}{\pgfqpoint{1.331cm}{1.836cm}}{\pgfqpoint{1.305cm}{1.862cm}}
\pgfpathcurveto{\pgfqpoint{1.28cm}{1.887cm}}{\pgfqpoint{1.245cm}{1.902cm}}{\pgfqpoint{1.209cm}{1.902cm}}
\pgfpathcurveto{\pgfqpoint{1.172cm}{1.902cm}}{\pgfqpoint{1.138cm}{1.887cm}}{\pgfqpoint{1.112cm}{1.862cm}}
\pgfpathcurveto{\pgfqpoint{1.087cm}{1.836cm}}{\pgfqpoint{1.072cm}{1.801cm}}{\pgfqpoint{1.072cm}{1.765cm}}
\pgfpathcurveto{\pgfqpoint{1.072cm}{1.728cm}}{\pgfqpoint{1.087cm}{1.694cm}}{\pgfqpoint{1.112cm}{1.668cm}}
\pgfpathcurveto{\pgfqpoint{1.138cm}{1.642cm}}{\pgfqpoint{1.172cm}{1.628cm}}{\pgfqpoint{1.209cm}{1.628cm}}
\pgfpathcurveto{\pgfqpoint{1.245cm}{1.628cm}}{\pgfqpoint{1.28cm}{1.642cm}}{\pgfqpoint{1.305cm}{1.668cm}}
\pgfpathcurveto{\pgfqpoint{1.331cm}{1.694cm}}{\pgfqpoint{1.345cm}{1.728cm}}{\pgfqpoint{1.345cm}{1.765cm}}
\pgfusepath{fill}
\begin{pgfscope}
\pgfsetdash{}{0cm}
\pgfsetlinewidth{0.818mm}
\pgfsetroundcap
\pgfsetroundjoin
\pgfsetmiterlimit{7.0}
\pgfpathmoveto{\pgfqpoint{0.682cm}{1.065cm}}
\pgfpathlineto{\pgfqpoint{1.246cm}{0.315cm}}
\pgfpathlineto{\pgfqpoint{1.811cm}{1.065cm}}
\pgfusepath{stroke}
\end{pgfscope}
\pgfpathmoveto{\pgfqpoint{1.948cm}{1.065cm}}
\pgfpathcurveto{\pgfqpoint{1.948cm}{1.101cm}}{\pgfqpoint{1.933cm}{1.136cm}}{\pgfqpoint{1.907cm}{1.162cm}}
\pgfpathcurveto{\pgfqpoint{1.882cm}{1.187cm}}{\pgfqpoint{1.847cm}{1.202cm}}{\pgfqpoint{1.811cm}{1.202cm}}
\pgfpathcurveto{\pgfqpoint{1.775cm}{1.202cm}}{\pgfqpoint{1.74cm}{1.187cm}}{\pgfqpoint{1.714cm}{1.162cm}}
\pgfpathcurveto{\pgfqpoint{1.689cm}{1.136cm}}{\pgfqpoint{1.674cm}{1.101cm}}{\pgfqpoint{1.674cm}{1.065cm}}
\pgfpathcurveto{\pgfqpoint{1.674cm}{1.029cm}}{\pgfqpoint{1.689cm}{0.994cm}}{\pgfqpoint{1.714cm}{0.968cm}}
\pgfpathcurveto{\pgfqpoint{1.74cm}{0.942cm}}{\pgfqpoint{1.775cm}{0.928cm}}{\pgfqpoint{1.811cm}{0.928cm}}
\pgfpathcurveto{\pgfqpoint{1.847cm}{0.928cm}}{\pgfqpoint{1.882cm}{0.942cm}}{\pgfqpoint{1.907cm}{0.968cm}}
\pgfpathcurveto{\pgfqpoint{1.933cm}{0.994cm}}{\pgfqpoint{1.948cm}{1.029cm}}{\pgfqpoint{1.948cm}{1.065cm}}
\pgfusepath{fill}
\begin{pgfscope}
\pgfsetdash{}{0cm}
\pgfsetlinewidth{0.818mm}
\pgfsetmiterlimit{7.0}
\pgfpathmoveto{\pgfqpoint{1.246cm}{0.315cm}}
\pgfpathlineto{\pgfqpoint{1.244cm}{1.061cm}}
\pgfusepath{stroke}
\end{pgfscope}
\pgfpathmoveto{\pgfqpoint{1.38cm}{1.065cm}}
\pgfpathcurveto{\pgfqpoint{1.38cm}{1.101cm}}{\pgfqpoint{1.366cm}{1.136cm}}{\pgfqpoint{1.34cm}{1.162cm}}
\pgfpathcurveto{\pgfqpoint{1.315cm}{1.187cm}}{\pgfqpoint{1.28cm}{1.202cm}}{\pgfqpoint{1.244cm}{1.202cm}}
\pgfpathcurveto{\pgfqpoint{1.207cm}{1.202cm}}{\pgfqpoint{1.173cm}{1.187cm}}{\pgfqpoint{1.147cm}{1.162cm}}
\pgfpathcurveto{\pgfqpoint{1.121cm}{1.136cm}}{\pgfqpoint{1.107cm}{1.101cm}}{\pgfqpoint{1.107cm}{1.065cm}}
\pgfpathcurveto{\pgfqpoint{1.107cm}{1.029cm}}{\pgfqpoint{1.121cm}{0.994cm}}{\pgfqpoint{1.147cm}{0.968cm}}
\pgfpathcurveto{\pgfqpoint{1.173cm}{0.942cm}}{\pgfqpoint{1.207cm}{0.928cm}}{\pgfqpoint{1.244cm}{0.928cm}}
\pgfpathcurveto{\pgfqpoint{1.28cm}{0.928cm}}{\pgfqpoint{1.315cm}{0.942cm}}{\pgfqpoint{1.34cm}{0.968cm}}
\pgfpathcurveto{\pgfqpoint{1.366cm}{0.994cm}}{\pgfqpoint{1.38cm}{1.029cm}}{\pgfqpoint{1.38cm}{1.065cm}}
\pgfusepath{fill}
\begin{pgfscope}
\pgfsetdash{}{0cm}
\pgfsetlinewidth{0.818mm}
\pgfsetmiterlimit{4.0}
\pgfpathmoveto{\pgfqpoint{1.383cm}{0.178cm}}
\pgfpathcurveto{\pgfqpoint{1.383cm}{0.214cm}}{\pgfqpoint{1.369cm}{0.249cm}}{\pgfqpoint{1.343cm}{0.275cm}}
\pgfpathcurveto{\pgfqpoint{1.317cm}{0.3cm}}{\pgfqpoint{1.283cm}{0.315cm}}{\pgfqpoint{1.246cm}{0.315cm}}
\pgfpathcurveto{\pgfqpoint{1.21cm}{0.315cm}}{\pgfqpoint{1.175cm}{0.3cm}}{\pgfqpoint{1.15cm}{0.275cm}}
\pgfpathcurveto{\pgfqpoint{1.124cm}{0.249cm}}{\pgfqpoint{1.11cm}{0.214cm}}{\pgfqpoint{1.11cm}{0.178cm}}
\pgfpathcurveto{\pgfqpoint{1.11cm}{0.141cm}}{\pgfqpoint{1.124cm}{0.107cm}}{\pgfqpoint{1.15cm}{0.081cm}}
\pgfpathcurveto{\pgfqpoint{1.175cm}{0.055cm}}{\pgfqpoint{1.21cm}{0.041cm}}{\pgfqpoint{1.246cm}{0.041cm}}
\pgfpathcurveto{\pgfqpoint{1.283cm}{0.041cm}}{\pgfqpoint{1.317cm}{0.055cm}}{\pgfqpoint{1.343cm}{0.081cm}}
\pgfpathcurveto{\pgfqpoint{1.369cm}{0.107cm}}{\pgfqpoint{1.383cm}{0.141cm}}{\pgfqpoint{1.383cm}{0.178cm}}
\pgfusepath{stroke}
\end{pgfscope}
\end{pgfscope}
\end{pgfscope}
\end{pgfscope}
\end{tikzpicture}}} \Join Y \rangle_{\varepsilon} | +
     \lambda^2 | \langle \rho^4 \phi_{\varepsilon},
     \tilde{X}_{\varepsilon}^{\!\resizebox{!}{.8em}{
\begin{tikzpicture}
\pgfpathmoveto{\pgfqpoint{0cm}{-0.035cm}}
\pgfpathlineto{\pgfqpoint{1.976cm}{-0.035cm}}
\pgfpathlineto{\pgfqpoint{1.976cm}{1.94cm}}
\pgfpathlineto{\pgfqpoint{0cm}{1.94cm}}
\pgfpathclose
\pgfusepath{clip}
\begin{pgfscope}
\begin{pgfscope}
\pgfpathmoveto{\pgfqpoint{0cm}{-0.035cm}}
\pgfpathlineto{\pgfqpoint{1.976cm}{-0.035cm}}
\pgfpathlineto{\pgfqpoint{1.976cm}{1.94cm}}
\pgfpathlineto{\pgfqpoint{0cm}{1.94cm}}
\pgfpathclose
\pgfusepath{clip}
\begin{pgfscope}
\begin{pgfscope}
\pgfsetdash{}{0cm}
\pgfsetlinewidth{0.818mm}
\pgfsetroundcap
\pgfsetroundjoin
\pgfsetmiterlimit{7.0}
\definecolor{eps2pgf_color}{gray}{0}\pgfsetstrokecolor{eps2pgf_color}\pgfsetfillcolor{eps2pgf_color}
\pgfpathmoveto{\pgfqpoint{0.117cm}{1.815cm}}
\pgfpathlineto{\pgfqpoint{0.682cm}{1.065cm}}
\pgfpathlineto{\pgfqpoint{1.246cm}{1.815cm}}
\pgfusepath{stroke}
\end{pgfscope}
\definecolor{eps2pgf_color}{gray}{0}\pgfsetstrokecolor{eps2pgf_color}\pgfsetfillcolor{eps2pgf_color}
\pgfpathmoveto{\pgfqpoint{0.273cm}{1.789cm}}
\pgfpathcurveto{\pgfqpoint{0.273cm}{1.825cm}}{\pgfqpoint{0.259cm}{1.86cm}}{\pgfqpoint{0.233cm}{1.886cm}}
\pgfpathcurveto{\pgfqpoint{0.207cm}{1.912cm}}{\pgfqpoint{0.173cm}{1.926cm}}{\pgfqpoint{0.137cm}{1.926cm}}
\pgfpathcurveto{\pgfqpoint{0.1cm}{1.926cm}}{\pgfqpoint{0.066cm}{1.912cm}}{\pgfqpoint{0.04cm}{1.886cm}}
\pgfpathcurveto{\pgfqpoint{0.014cm}{1.86cm}}{\pgfqpoint{0cm}{1.825cm}}{\pgfqpoint{0cm}{1.789cm}}
\pgfpathcurveto{\pgfqpoint{0cm}{1.753cm}}{\pgfqpoint{0.014cm}{1.718cm}}{\pgfqpoint{0.04cm}{1.692cm}}
\pgfpathcurveto{\pgfqpoint{0.066cm}{1.667cm}}{\pgfqpoint{0.1cm}{1.652cm}}{\pgfqpoint{0.137cm}{1.652cm}}
\pgfpathcurveto{\pgfqpoint{0.173cm}{1.652cm}}{\pgfqpoint{0.207cm}{1.667cm}}{\pgfqpoint{0.233cm}{1.692cm}}
\pgfpathcurveto{\pgfqpoint{0.259cm}{1.718cm}}{\pgfqpoint{0.273cm}{1.753cm}}{\pgfqpoint{0.273cm}{1.789cm}}
\pgfusepath{fill}
\pgfpathmoveto{\pgfqpoint{1.345cm}{1.765cm}}
\pgfpathcurveto{\pgfqpoint{1.345cm}{1.801cm}}{\pgfqpoint{1.331cm}{1.836cm}}{\pgfqpoint{1.305cm}{1.862cm}}
\pgfpathcurveto{\pgfqpoint{1.28cm}{1.887cm}}{\pgfqpoint{1.245cm}{1.902cm}}{\pgfqpoint{1.209cm}{1.902cm}}
\pgfpathcurveto{\pgfqpoint{1.172cm}{1.902cm}}{\pgfqpoint{1.138cm}{1.887cm}}{\pgfqpoint{1.112cm}{1.862cm}}
\pgfpathcurveto{\pgfqpoint{1.087cm}{1.836cm}}{\pgfqpoint{1.072cm}{1.801cm}}{\pgfqpoint{1.072cm}{1.765cm}}
\pgfpathcurveto{\pgfqpoint{1.072cm}{1.728cm}}{\pgfqpoint{1.087cm}{1.694cm}}{\pgfqpoint{1.112cm}{1.668cm}}
\pgfpathcurveto{\pgfqpoint{1.138cm}{1.642cm}}{\pgfqpoint{1.172cm}{1.628cm}}{\pgfqpoint{1.209cm}{1.628cm}}
\pgfpathcurveto{\pgfqpoint{1.245cm}{1.628cm}}{\pgfqpoint{1.28cm}{1.642cm}}{\pgfqpoint{1.305cm}{1.668cm}}
\pgfpathcurveto{\pgfqpoint{1.331cm}{1.694cm}}{\pgfqpoint{1.345cm}{1.728cm}}{\pgfqpoint{1.345cm}{1.765cm}}
\pgfusepath{fill}
\begin{pgfscope}
\pgfsetdash{}{0cm}
\pgfsetlinewidth{0.818mm}
\pgfsetroundcap
\pgfsetroundjoin
\pgfsetmiterlimit{7.0}
\pgfpathmoveto{\pgfqpoint{0.682cm}{1.065cm}}
\pgfpathlineto{\pgfqpoint{1.246cm}{0.315cm}}
\pgfpathlineto{\pgfqpoint{1.811cm}{1.065cm}}
\pgfusepath{stroke}
\end{pgfscope}
\pgfpathmoveto{\pgfqpoint{1.948cm}{1.065cm}}
\pgfpathcurveto{\pgfqpoint{1.948cm}{1.101cm}}{\pgfqpoint{1.933cm}{1.136cm}}{\pgfqpoint{1.907cm}{1.162cm}}
\pgfpathcurveto{\pgfqpoint{1.882cm}{1.187cm}}{\pgfqpoint{1.847cm}{1.202cm}}{\pgfqpoint{1.811cm}{1.202cm}}
\pgfpathcurveto{\pgfqpoint{1.775cm}{1.202cm}}{\pgfqpoint{1.74cm}{1.187cm}}{\pgfqpoint{1.714cm}{1.162cm}}
\pgfpathcurveto{\pgfqpoint{1.689cm}{1.136cm}}{\pgfqpoint{1.674cm}{1.101cm}}{\pgfqpoint{1.674cm}{1.065cm}}
\pgfpathcurveto{\pgfqpoint{1.674cm}{1.029cm}}{\pgfqpoint{1.689cm}{0.994cm}}{\pgfqpoint{1.714cm}{0.968cm}}
\pgfpathcurveto{\pgfqpoint{1.74cm}{0.942cm}}{\pgfqpoint{1.775cm}{0.928cm}}{\pgfqpoint{1.811cm}{0.928cm}}
\pgfpathcurveto{\pgfqpoint{1.847cm}{0.928cm}}{\pgfqpoint{1.882cm}{0.942cm}}{\pgfqpoint{1.907cm}{0.968cm}}
\pgfpathcurveto{\pgfqpoint{1.933cm}{0.994cm}}{\pgfqpoint{1.948cm}{1.029cm}}{\pgfqpoint{1.948cm}{1.065cm}}
\pgfusepath{fill}
\begin{pgfscope}
\pgfsetdash{}{0cm}
\pgfsetlinewidth{0.818mm}
\pgfsetmiterlimit{7.0}
\pgfpathmoveto{\pgfqpoint{1.246cm}{0.315cm}}
\pgfpathlineto{\pgfqpoint{1.244cm}{1.061cm}}
\pgfusepath{stroke}
\end{pgfscope}
\pgfpathmoveto{\pgfqpoint{1.38cm}{1.065cm}}
\pgfpathcurveto{\pgfqpoint{1.38cm}{1.101cm}}{\pgfqpoint{1.366cm}{1.136cm}}{\pgfqpoint{1.34cm}{1.162cm}}
\pgfpathcurveto{\pgfqpoint{1.315cm}{1.187cm}}{\pgfqpoint{1.28cm}{1.202cm}}{\pgfqpoint{1.244cm}{1.202cm}}
\pgfpathcurveto{\pgfqpoint{1.207cm}{1.202cm}}{\pgfqpoint{1.173cm}{1.187cm}}{\pgfqpoint{1.147cm}{1.162cm}}
\pgfpathcurveto{\pgfqpoint{1.121cm}{1.136cm}}{\pgfqpoint{1.107cm}{1.101cm}}{\pgfqpoint{1.107cm}{1.065cm}}
\pgfpathcurveto{\pgfqpoint{1.107cm}{1.029cm}}{\pgfqpoint{1.121cm}{0.994cm}}{\pgfqpoint{1.147cm}{0.968cm}}
\pgfpathcurveto{\pgfqpoint{1.173cm}{0.942cm}}{\pgfqpoint{1.207cm}{0.928cm}}{\pgfqpoint{1.244cm}{0.928cm}}
\pgfpathcurveto{\pgfqpoint{1.28cm}{0.928cm}}{\pgfqpoint{1.315cm}{0.942cm}}{\pgfqpoint{1.34cm}{0.968cm}}
\pgfpathcurveto{\pgfqpoint{1.366cm}{0.994cm}}{\pgfqpoint{1.38cm}{1.029cm}}{\pgfqpoint{1.38cm}{1.065cm}}
\pgfusepath{fill}
\begin{pgfscope}
\pgfsetdash{}{0cm}
\pgfsetlinewidth{0.818mm}
\pgfsetmiterlimit{4.0}
\pgfpathmoveto{\pgfqpoint{1.383cm}{0.178cm}}
\pgfpathcurveto{\pgfqpoint{1.383cm}{0.214cm}}{\pgfqpoint{1.369cm}{0.249cm}}{\pgfqpoint{1.343cm}{0.275cm}}
\pgfpathcurveto{\pgfqpoint{1.317cm}{0.3cm}}{\pgfqpoint{1.283cm}{0.315cm}}{\pgfqpoint{1.246cm}{0.315cm}}
\pgfpathcurveto{\pgfqpoint{1.21cm}{0.315cm}}{\pgfqpoint{1.175cm}{0.3cm}}{\pgfqpoint{1.15cm}{0.275cm}}
\pgfpathcurveto{\pgfqpoint{1.124cm}{0.249cm}}{\pgfqpoint{1.11cm}{0.214cm}}{\pgfqpoint{1.11cm}{0.178cm}}
\pgfpathcurveto{\pgfqpoint{1.11cm}{0.141cm}}{\pgfqpoint{1.124cm}{0.107cm}}{\pgfqpoint{1.15cm}{0.081cm}}
\pgfpathcurveto{\pgfqpoint{1.175cm}{0.055cm}}{\pgfqpoint{1.21cm}{0.041cm}}{\pgfqpoint{1.246cm}{0.041cm}}
\pgfpathcurveto{\pgfqpoint{1.283cm}{0.041cm}}{\pgfqpoint{1.317cm}{0.055cm}}{\pgfqpoint{1.343cm}{0.081cm}}
\pgfpathcurveto{\pgfqpoint{1.369cm}{0.107cm}}{\pgfqpoint{1.383cm}{0.141cm}}{\pgfqpoint{1.383cm}{0.178cm}}
\pgfusepath{stroke}
\end{pgfscope}
\end{pgfscope}
\end{pgfscope}
\end{pgfscope}
\end{tikzpicture}}} \circ Y \rangle_{\varepsilon} | \]
  where, for $\theta = \frac{1 - 4 \kappa}{1 - 2 \kappa}$, we bound
  \[ \lambda^2 | \langle \rho^4 \phi_{\varepsilon},
     \tilde{X}_{\varepsilon}^{\!\resizebox{!}{.8em}{
\begin{tikzpicture}
\pgfpathmoveto{\pgfqpoint{0cm}{-0.035cm}}
\pgfpathlineto{\pgfqpoint{1.976cm}{-0.035cm}}
\pgfpathlineto{\pgfqpoint{1.976cm}{1.94cm}}
\pgfpathlineto{\pgfqpoint{0cm}{1.94cm}}
\pgfpathclose
\pgfusepath{clip}
\begin{pgfscope}
\begin{pgfscope}
\pgfpathmoveto{\pgfqpoint{0cm}{-0.035cm}}
\pgfpathlineto{\pgfqpoint{1.976cm}{-0.035cm}}
\pgfpathlineto{\pgfqpoint{1.976cm}{1.94cm}}
\pgfpathlineto{\pgfqpoint{0cm}{1.94cm}}
\pgfpathclose
\pgfusepath{clip}
\begin{pgfscope}
\begin{pgfscope}
\pgfsetdash{}{0cm}
\pgfsetlinewidth{0.818mm}
\pgfsetroundcap
\pgfsetroundjoin
\pgfsetmiterlimit{7.0}
\definecolor{eps2pgf_color}{gray}{0}\pgfsetstrokecolor{eps2pgf_color}\pgfsetfillcolor{eps2pgf_color}
\pgfpathmoveto{\pgfqpoint{0.117cm}{1.815cm}}
\pgfpathlineto{\pgfqpoint{0.682cm}{1.065cm}}
\pgfpathlineto{\pgfqpoint{1.246cm}{1.815cm}}
\pgfusepath{stroke}
\end{pgfscope}
\definecolor{eps2pgf_color}{gray}{0}\pgfsetstrokecolor{eps2pgf_color}\pgfsetfillcolor{eps2pgf_color}
\pgfpathmoveto{\pgfqpoint{0.273cm}{1.789cm}}
\pgfpathcurveto{\pgfqpoint{0.273cm}{1.825cm}}{\pgfqpoint{0.259cm}{1.86cm}}{\pgfqpoint{0.233cm}{1.886cm}}
\pgfpathcurveto{\pgfqpoint{0.207cm}{1.912cm}}{\pgfqpoint{0.173cm}{1.926cm}}{\pgfqpoint{0.137cm}{1.926cm}}
\pgfpathcurveto{\pgfqpoint{0.1cm}{1.926cm}}{\pgfqpoint{0.066cm}{1.912cm}}{\pgfqpoint{0.04cm}{1.886cm}}
\pgfpathcurveto{\pgfqpoint{0.014cm}{1.86cm}}{\pgfqpoint{0cm}{1.825cm}}{\pgfqpoint{0cm}{1.789cm}}
\pgfpathcurveto{\pgfqpoint{0cm}{1.753cm}}{\pgfqpoint{0.014cm}{1.718cm}}{\pgfqpoint{0.04cm}{1.692cm}}
\pgfpathcurveto{\pgfqpoint{0.066cm}{1.667cm}}{\pgfqpoint{0.1cm}{1.652cm}}{\pgfqpoint{0.137cm}{1.652cm}}
\pgfpathcurveto{\pgfqpoint{0.173cm}{1.652cm}}{\pgfqpoint{0.207cm}{1.667cm}}{\pgfqpoint{0.233cm}{1.692cm}}
\pgfpathcurveto{\pgfqpoint{0.259cm}{1.718cm}}{\pgfqpoint{0.273cm}{1.753cm}}{\pgfqpoint{0.273cm}{1.789cm}}
\pgfusepath{fill}
\pgfpathmoveto{\pgfqpoint{1.345cm}{1.765cm}}
\pgfpathcurveto{\pgfqpoint{1.345cm}{1.801cm}}{\pgfqpoint{1.331cm}{1.836cm}}{\pgfqpoint{1.305cm}{1.862cm}}
\pgfpathcurveto{\pgfqpoint{1.28cm}{1.887cm}}{\pgfqpoint{1.245cm}{1.902cm}}{\pgfqpoint{1.209cm}{1.902cm}}
\pgfpathcurveto{\pgfqpoint{1.172cm}{1.902cm}}{\pgfqpoint{1.138cm}{1.887cm}}{\pgfqpoint{1.112cm}{1.862cm}}
\pgfpathcurveto{\pgfqpoint{1.087cm}{1.836cm}}{\pgfqpoint{1.072cm}{1.801cm}}{\pgfqpoint{1.072cm}{1.765cm}}
\pgfpathcurveto{\pgfqpoint{1.072cm}{1.728cm}}{\pgfqpoint{1.087cm}{1.694cm}}{\pgfqpoint{1.112cm}{1.668cm}}
\pgfpathcurveto{\pgfqpoint{1.138cm}{1.642cm}}{\pgfqpoint{1.172cm}{1.628cm}}{\pgfqpoint{1.209cm}{1.628cm}}
\pgfpathcurveto{\pgfqpoint{1.245cm}{1.628cm}}{\pgfqpoint{1.28cm}{1.642cm}}{\pgfqpoint{1.305cm}{1.668cm}}
\pgfpathcurveto{\pgfqpoint{1.331cm}{1.694cm}}{\pgfqpoint{1.345cm}{1.728cm}}{\pgfqpoint{1.345cm}{1.765cm}}
\pgfusepath{fill}
\begin{pgfscope}
\pgfsetdash{}{0cm}
\pgfsetlinewidth{0.818mm}
\pgfsetroundcap
\pgfsetroundjoin
\pgfsetmiterlimit{7.0}
\pgfpathmoveto{\pgfqpoint{0.682cm}{1.065cm}}
\pgfpathlineto{\pgfqpoint{1.246cm}{0.315cm}}
\pgfpathlineto{\pgfqpoint{1.811cm}{1.065cm}}
\pgfusepath{stroke}
\end{pgfscope}
\pgfpathmoveto{\pgfqpoint{1.948cm}{1.065cm}}
\pgfpathcurveto{\pgfqpoint{1.948cm}{1.101cm}}{\pgfqpoint{1.933cm}{1.136cm}}{\pgfqpoint{1.907cm}{1.162cm}}
\pgfpathcurveto{\pgfqpoint{1.882cm}{1.187cm}}{\pgfqpoint{1.847cm}{1.202cm}}{\pgfqpoint{1.811cm}{1.202cm}}
\pgfpathcurveto{\pgfqpoint{1.775cm}{1.202cm}}{\pgfqpoint{1.74cm}{1.187cm}}{\pgfqpoint{1.714cm}{1.162cm}}
\pgfpathcurveto{\pgfqpoint{1.689cm}{1.136cm}}{\pgfqpoint{1.674cm}{1.101cm}}{\pgfqpoint{1.674cm}{1.065cm}}
\pgfpathcurveto{\pgfqpoint{1.674cm}{1.029cm}}{\pgfqpoint{1.689cm}{0.994cm}}{\pgfqpoint{1.714cm}{0.968cm}}
\pgfpathcurveto{\pgfqpoint{1.74cm}{0.942cm}}{\pgfqpoint{1.775cm}{0.928cm}}{\pgfqpoint{1.811cm}{0.928cm}}
\pgfpathcurveto{\pgfqpoint{1.847cm}{0.928cm}}{\pgfqpoint{1.882cm}{0.942cm}}{\pgfqpoint{1.907cm}{0.968cm}}
\pgfpathcurveto{\pgfqpoint{1.933cm}{0.994cm}}{\pgfqpoint{1.948cm}{1.029cm}}{\pgfqpoint{1.948cm}{1.065cm}}
\pgfusepath{fill}
\begin{pgfscope}
\pgfsetdash{}{0cm}
\pgfsetlinewidth{0.818mm}
\pgfsetmiterlimit{7.0}
\pgfpathmoveto{\pgfqpoint{1.246cm}{0.315cm}}
\pgfpathlineto{\pgfqpoint{1.244cm}{1.061cm}}
\pgfusepath{stroke}
\end{pgfscope}
\pgfpathmoveto{\pgfqpoint{1.38cm}{1.065cm}}
\pgfpathcurveto{\pgfqpoint{1.38cm}{1.101cm}}{\pgfqpoint{1.366cm}{1.136cm}}{\pgfqpoint{1.34cm}{1.162cm}}
\pgfpathcurveto{\pgfqpoint{1.315cm}{1.187cm}}{\pgfqpoint{1.28cm}{1.202cm}}{\pgfqpoint{1.244cm}{1.202cm}}
\pgfpathcurveto{\pgfqpoint{1.207cm}{1.202cm}}{\pgfqpoint{1.173cm}{1.187cm}}{\pgfqpoint{1.147cm}{1.162cm}}
\pgfpathcurveto{\pgfqpoint{1.121cm}{1.136cm}}{\pgfqpoint{1.107cm}{1.101cm}}{\pgfqpoint{1.107cm}{1.065cm}}
\pgfpathcurveto{\pgfqpoint{1.107cm}{1.029cm}}{\pgfqpoint{1.121cm}{0.994cm}}{\pgfqpoint{1.147cm}{0.968cm}}
\pgfpathcurveto{\pgfqpoint{1.173cm}{0.942cm}}{\pgfqpoint{1.207cm}{0.928cm}}{\pgfqpoint{1.244cm}{0.928cm}}
\pgfpathcurveto{\pgfqpoint{1.28cm}{0.928cm}}{\pgfqpoint{1.315cm}{0.942cm}}{\pgfqpoint{1.34cm}{0.968cm}}
\pgfpathcurveto{\pgfqpoint{1.366cm}{0.994cm}}{\pgfqpoint{1.38cm}{1.029cm}}{\pgfqpoint{1.38cm}{1.065cm}}
\pgfusepath{fill}
\begin{pgfscope}
\pgfsetdash{}{0cm}
\pgfsetlinewidth{0.818mm}
\pgfsetmiterlimit{4.0}
\pgfpathmoveto{\pgfqpoint{1.383cm}{0.178cm}}
\pgfpathcurveto{\pgfqpoint{1.383cm}{0.214cm}}{\pgfqpoint{1.369cm}{0.249cm}}{\pgfqpoint{1.343cm}{0.275cm}}
\pgfpathcurveto{\pgfqpoint{1.317cm}{0.3cm}}{\pgfqpoint{1.283cm}{0.315cm}}{\pgfqpoint{1.246cm}{0.315cm}}
\pgfpathcurveto{\pgfqpoint{1.21cm}{0.315cm}}{\pgfqpoint{1.175cm}{0.3cm}}{\pgfqpoint{1.15cm}{0.275cm}}
\pgfpathcurveto{\pgfqpoint{1.124cm}{0.249cm}}{\pgfqpoint{1.11cm}{0.214cm}}{\pgfqpoint{1.11cm}{0.178cm}}
\pgfpathcurveto{\pgfqpoint{1.11cm}{0.141cm}}{\pgfqpoint{1.124cm}{0.107cm}}{\pgfqpoint{1.15cm}{0.081cm}}
\pgfpathcurveto{\pgfqpoint{1.175cm}{0.055cm}}{\pgfqpoint{1.21cm}{0.041cm}}{\pgfqpoint{1.246cm}{0.041cm}}
\pgfpathcurveto{\pgfqpoint{1.283cm}{0.041cm}}{\pgfqpoint{1.317cm}{0.055cm}}{\pgfqpoint{1.343cm}{0.081cm}}
\pgfpathcurveto{\pgfqpoint{1.369cm}{0.107cm}}{\pgfqpoint{1.383cm}{0.141cm}}{\pgfqpoint{1.383cm}{0.178cm}}
\pgfusepath{stroke}
\end{pgfscope}
\end{pgfscope}
\end{pgfscope}
\end{pgfscope}
\end{tikzpicture}}} \Join Y_{\varepsilon}
     \rangle_{\varepsilon} | \lesssim \lambda^2 \| \rho^{4 - 2 \sigma}
     \phi_{\varepsilon} \|_{B^{\kappa, \varepsilon}_{1, 1}} \| \rho^{2 \sigma}
     \tilde{X}_{\varepsilon}^{\!\resizebox{!}{.8em}{
\begin{tikzpicture}
\pgfpathmoveto{\pgfqpoint{0cm}{-0.035cm}}
\pgfpathlineto{\pgfqpoint{1.976cm}{-0.035cm}}
\pgfpathlineto{\pgfqpoint{1.976cm}{1.94cm}}
\pgfpathlineto{\pgfqpoint{0cm}{1.94cm}}
\pgfpathclose
\pgfusepath{clip}
\begin{pgfscope}
\begin{pgfscope}
\pgfpathmoveto{\pgfqpoint{0cm}{-0.035cm}}
\pgfpathlineto{\pgfqpoint{1.976cm}{-0.035cm}}
\pgfpathlineto{\pgfqpoint{1.976cm}{1.94cm}}
\pgfpathlineto{\pgfqpoint{0cm}{1.94cm}}
\pgfpathclose
\pgfusepath{clip}
\begin{pgfscope}
\begin{pgfscope}
\pgfsetdash{}{0cm}
\pgfsetlinewidth{0.818mm}
\pgfsetroundcap
\pgfsetroundjoin
\pgfsetmiterlimit{7.0}
\definecolor{eps2pgf_color}{gray}{0}\pgfsetstrokecolor{eps2pgf_color}\pgfsetfillcolor{eps2pgf_color}
\pgfpathmoveto{\pgfqpoint{0.117cm}{1.815cm}}
\pgfpathlineto{\pgfqpoint{0.682cm}{1.065cm}}
\pgfpathlineto{\pgfqpoint{1.246cm}{1.815cm}}
\pgfusepath{stroke}
\end{pgfscope}
\definecolor{eps2pgf_color}{gray}{0}\pgfsetstrokecolor{eps2pgf_color}\pgfsetfillcolor{eps2pgf_color}
\pgfpathmoveto{\pgfqpoint{0.273cm}{1.789cm}}
\pgfpathcurveto{\pgfqpoint{0.273cm}{1.825cm}}{\pgfqpoint{0.259cm}{1.86cm}}{\pgfqpoint{0.233cm}{1.886cm}}
\pgfpathcurveto{\pgfqpoint{0.207cm}{1.912cm}}{\pgfqpoint{0.173cm}{1.926cm}}{\pgfqpoint{0.137cm}{1.926cm}}
\pgfpathcurveto{\pgfqpoint{0.1cm}{1.926cm}}{\pgfqpoint{0.066cm}{1.912cm}}{\pgfqpoint{0.04cm}{1.886cm}}
\pgfpathcurveto{\pgfqpoint{0.014cm}{1.86cm}}{\pgfqpoint{0cm}{1.825cm}}{\pgfqpoint{0cm}{1.789cm}}
\pgfpathcurveto{\pgfqpoint{0cm}{1.753cm}}{\pgfqpoint{0.014cm}{1.718cm}}{\pgfqpoint{0.04cm}{1.692cm}}
\pgfpathcurveto{\pgfqpoint{0.066cm}{1.667cm}}{\pgfqpoint{0.1cm}{1.652cm}}{\pgfqpoint{0.137cm}{1.652cm}}
\pgfpathcurveto{\pgfqpoint{0.173cm}{1.652cm}}{\pgfqpoint{0.207cm}{1.667cm}}{\pgfqpoint{0.233cm}{1.692cm}}
\pgfpathcurveto{\pgfqpoint{0.259cm}{1.718cm}}{\pgfqpoint{0.273cm}{1.753cm}}{\pgfqpoint{0.273cm}{1.789cm}}
\pgfusepath{fill}
\pgfpathmoveto{\pgfqpoint{1.345cm}{1.765cm}}
\pgfpathcurveto{\pgfqpoint{1.345cm}{1.801cm}}{\pgfqpoint{1.331cm}{1.836cm}}{\pgfqpoint{1.305cm}{1.862cm}}
\pgfpathcurveto{\pgfqpoint{1.28cm}{1.887cm}}{\pgfqpoint{1.245cm}{1.902cm}}{\pgfqpoint{1.209cm}{1.902cm}}
\pgfpathcurveto{\pgfqpoint{1.172cm}{1.902cm}}{\pgfqpoint{1.138cm}{1.887cm}}{\pgfqpoint{1.112cm}{1.862cm}}
\pgfpathcurveto{\pgfqpoint{1.087cm}{1.836cm}}{\pgfqpoint{1.072cm}{1.801cm}}{\pgfqpoint{1.072cm}{1.765cm}}
\pgfpathcurveto{\pgfqpoint{1.072cm}{1.728cm}}{\pgfqpoint{1.087cm}{1.694cm}}{\pgfqpoint{1.112cm}{1.668cm}}
\pgfpathcurveto{\pgfqpoint{1.138cm}{1.642cm}}{\pgfqpoint{1.172cm}{1.628cm}}{\pgfqpoint{1.209cm}{1.628cm}}
\pgfpathcurveto{\pgfqpoint{1.245cm}{1.628cm}}{\pgfqpoint{1.28cm}{1.642cm}}{\pgfqpoint{1.305cm}{1.668cm}}
\pgfpathcurveto{\pgfqpoint{1.331cm}{1.694cm}}{\pgfqpoint{1.345cm}{1.728cm}}{\pgfqpoint{1.345cm}{1.765cm}}
\pgfusepath{fill}
\begin{pgfscope}
\pgfsetdash{}{0cm}
\pgfsetlinewidth{0.818mm}
\pgfsetroundcap
\pgfsetroundjoin
\pgfsetmiterlimit{7.0}
\pgfpathmoveto{\pgfqpoint{0.682cm}{1.065cm}}
\pgfpathlineto{\pgfqpoint{1.246cm}{0.315cm}}
\pgfpathlineto{\pgfqpoint{1.811cm}{1.065cm}}
\pgfusepath{stroke}
\end{pgfscope}
\pgfpathmoveto{\pgfqpoint{1.948cm}{1.065cm}}
\pgfpathcurveto{\pgfqpoint{1.948cm}{1.101cm}}{\pgfqpoint{1.933cm}{1.136cm}}{\pgfqpoint{1.907cm}{1.162cm}}
\pgfpathcurveto{\pgfqpoint{1.882cm}{1.187cm}}{\pgfqpoint{1.847cm}{1.202cm}}{\pgfqpoint{1.811cm}{1.202cm}}
\pgfpathcurveto{\pgfqpoint{1.775cm}{1.202cm}}{\pgfqpoint{1.74cm}{1.187cm}}{\pgfqpoint{1.714cm}{1.162cm}}
\pgfpathcurveto{\pgfqpoint{1.689cm}{1.136cm}}{\pgfqpoint{1.674cm}{1.101cm}}{\pgfqpoint{1.674cm}{1.065cm}}
\pgfpathcurveto{\pgfqpoint{1.674cm}{1.029cm}}{\pgfqpoint{1.689cm}{0.994cm}}{\pgfqpoint{1.714cm}{0.968cm}}
\pgfpathcurveto{\pgfqpoint{1.74cm}{0.942cm}}{\pgfqpoint{1.775cm}{0.928cm}}{\pgfqpoint{1.811cm}{0.928cm}}
\pgfpathcurveto{\pgfqpoint{1.847cm}{0.928cm}}{\pgfqpoint{1.882cm}{0.942cm}}{\pgfqpoint{1.907cm}{0.968cm}}
\pgfpathcurveto{\pgfqpoint{1.933cm}{0.994cm}}{\pgfqpoint{1.948cm}{1.029cm}}{\pgfqpoint{1.948cm}{1.065cm}}
\pgfusepath{fill}
\begin{pgfscope}
\pgfsetdash{}{0cm}
\pgfsetlinewidth{0.818mm}
\pgfsetmiterlimit{7.0}
\pgfpathmoveto{\pgfqpoint{1.246cm}{0.315cm}}
\pgfpathlineto{\pgfqpoint{1.244cm}{1.061cm}}
\pgfusepath{stroke}
\end{pgfscope}
\pgfpathmoveto{\pgfqpoint{1.38cm}{1.065cm}}
\pgfpathcurveto{\pgfqpoint{1.38cm}{1.101cm}}{\pgfqpoint{1.366cm}{1.136cm}}{\pgfqpoint{1.34cm}{1.162cm}}
\pgfpathcurveto{\pgfqpoint{1.315cm}{1.187cm}}{\pgfqpoint{1.28cm}{1.202cm}}{\pgfqpoint{1.244cm}{1.202cm}}
\pgfpathcurveto{\pgfqpoint{1.207cm}{1.202cm}}{\pgfqpoint{1.173cm}{1.187cm}}{\pgfqpoint{1.147cm}{1.162cm}}
\pgfpathcurveto{\pgfqpoint{1.121cm}{1.136cm}}{\pgfqpoint{1.107cm}{1.101cm}}{\pgfqpoint{1.107cm}{1.065cm}}
\pgfpathcurveto{\pgfqpoint{1.107cm}{1.029cm}}{\pgfqpoint{1.121cm}{0.994cm}}{\pgfqpoint{1.147cm}{0.968cm}}
\pgfpathcurveto{\pgfqpoint{1.173cm}{0.942cm}}{\pgfqpoint{1.207cm}{0.928cm}}{\pgfqpoint{1.244cm}{0.928cm}}
\pgfpathcurveto{\pgfqpoint{1.28cm}{0.928cm}}{\pgfqpoint{1.315cm}{0.942cm}}{\pgfqpoint{1.34cm}{0.968cm}}
\pgfpathcurveto{\pgfqpoint{1.366cm}{0.994cm}}{\pgfqpoint{1.38cm}{1.029cm}}{\pgfqpoint{1.38cm}{1.065cm}}
\pgfusepath{fill}
\begin{pgfscope}
\pgfsetdash{}{0cm}
\pgfsetlinewidth{0.818mm}
\pgfsetmiterlimit{4.0}
\pgfpathmoveto{\pgfqpoint{1.383cm}{0.178cm}}
\pgfpathcurveto{\pgfqpoint{1.383cm}{0.214cm}}{\pgfqpoint{1.369cm}{0.249cm}}{\pgfqpoint{1.343cm}{0.275cm}}
\pgfpathcurveto{\pgfqpoint{1.317cm}{0.3cm}}{\pgfqpoint{1.283cm}{0.315cm}}{\pgfqpoint{1.246cm}{0.315cm}}
\pgfpathcurveto{\pgfqpoint{1.21cm}{0.315cm}}{\pgfqpoint{1.175cm}{0.3cm}}{\pgfqpoint{1.15cm}{0.275cm}}
\pgfpathcurveto{\pgfqpoint{1.124cm}{0.249cm}}{\pgfqpoint{1.11cm}{0.214cm}}{\pgfqpoint{1.11cm}{0.178cm}}
\pgfpathcurveto{\pgfqpoint{1.11cm}{0.141cm}}{\pgfqpoint{1.124cm}{0.107cm}}{\pgfqpoint{1.15cm}{0.081cm}}
\pgfpathcurveto{\pgfqpoint{1.175cm}{0.055cm}}{\pgfqpoint{1.21cm}{0.041cm}}{\pgfqpoint{1.246cm}{0.041cm}}
\pgfpathcurveto{\pgfqpoint{1.283cm}{0.041cm}}{\pgfqpoint{1.317cm}{0.055cm}}{\pgfqpoint{1.343cm}{0.081cm}}
\pgfpathcurveto{\pgfqpoint{1.369cm}{0.107cm}}{\pgfqpoint{1.383cm}{0.141cm}}{\pgfqpoint{1.383cm}{0.178cm}}
\pgfusepath{stroke}
\end{pgfscope}
\end{pgfscope}
\end{pgfscope}
\end{pgfscope}
\end{tikzpicture}}} \Join Y_{\varepsilon}
     \|_{\mathscr{C} \hspace{.1em}^{- \kappa, \varepsilon}} \]
  \[ \lesssim \lambda^2 \| \rho \phi_{\varepsilon} \|^{\theta}_{L^{4,
     \varepsilon}} \| \rho^2 \phi_{\varepsilon} \|^{1 - \theta}_{H^{1 - 2
     \kappa, \varepsilon}} \| \rho^{\sigma}
     \tilde{X}_{\varepsilon}^{\!\resizebox{!}{.8em}{
\begin{tikzpicture}
\pgfpathmoveto{\pgfqpoint{0cm}{-0.035cm}}
\pgfpathlineto{\pgfqpoint{1.976cm}{-0.035cm}}
\pgfpathlineto{\pgfqpoint{1.976cm}{1.94cm}}
\pgfpathlineto{\pgfqpoint{0cm}{1.94cm}}
\pgfpathclose
\pgfusepath{clip}
\begin{pgfscope}
\begin{pgfscope}
\pgfpathmoveto{\pgfqpoint{0cm}{-0.035cm}}
\pgfpathlineto{\pgfqpoint{1.976cm}{-0.035cm}}
\pgfpathlineto{\pgfqpoint{1.976cm}{1.94cm}}
\pgfpathlineto{\pgfqpoint{0cm}{1.94cm}}
\pgfpathclose
\pgfusepath{clip}
\begin{pgfscope}
\begin{pgfscope}
\pgfsetdash{}{0cm}
\pgfsetlinewidth{0.818mm}
\pgfsetroundcap
\pgfsetroundjoin
\pgfsetmiterlimit{7.0}
\definecolor{eps2pgf_color}{gray}{0}\pgfsetstrokecolor{eps2pgf_color}\pgfsetfillcolor{eps2pgf_color}
\pgfpathmoveto{\pgfqpoint{0.117cm}{1.815cm}}
\pgfpathlineto{\pgfqpoint{0.682cm}{1.065cm}}
\pgfpathlineto{\pgfqpoint{1.246cm}{1.815cm}}
\pgfusepath{stroke}
\end{pgfscope}
\definecolor{eps2pgf_color}{gray}{0}\pgfsetstrokecolor{eps2pgf_color}\pgfsetfillcolor{eps2pgf_color}
\pgfpathmoveto{\pgfqpoint{0.273cm}{1.789cm}}
\pgfpathcurveto{\pgfqpoint{0.273cm}{1.825cm}}{\pgfqpoint{0.259cm}{1.86cm}}{\pgfqpoint{0.233cm}{1.886cm}}
\pgfpathcurveto{\pgfqpoint{0.207cm}{1.912cm}}{\pgfqpoint{0.173cm}{1.926cm}}{\pgfqpoint{0.137cm}{1.926cm}}
\pgfpathcurveto{\pgfqpoint{0.1cm}{1.926cm}}{\pgfqpoint{0.066cm}{1.912cm}}{\pgfqpoint{0.04cm}{1.886cm}}
\pgfpathcurveto{\pgfqpoint{0.014cm}{1.86cm}}{\pgfqpoint{0cm}{1.825cm}}{\pgfqpoint{0cm}{1.789cm}}
\pgfpathcurveto{\pgfqpoint{0cm}{1.753cm}}{\pgfqpoint{0.014cm}{1.718cm}}{\pgfqpoint{0.04cm}{1.692cm}}
\pgfpathcurveto{\pgfqpoint{0.066cm}{1.667cm}}{\pgfqpoint{0.1cm}{1.652cm}}{\pgfqpoint{0.137cm}{1.652cm}}
\pgfpathcurveto{\pgfqpoint{0.173cm}{1.652cm}}{\pgfqpoint{0.207cm}{1.667cm}}{\pgfqpoint{0.233cm}{1.692cm}}
\pgfpathcurveto{\pgfqpoint{0.259cm}{1.718cm}}{\pgfqpoint{0.273cm}{1.753cm}}{\pgfqpoint{0.273cm}{1.789cm}}
\pgfusepath{fill}
\pgfpathmoveto{\pgfqpoint{1.345cm}{1.765cm}}
\pgfpathcurveto{\pgfqpoint{1.345cm}{1.801cm}}{\pgfqpoint{1.331cm}{1.836cm}}{\pgfqpoint{1.305cm}{1.862cm}}
\pgfpathcurveto{\pgfqpoint{1.28cm}{1.887cm}}{\pgfqpoint{1.245cm}{1.902cm}}{\pgfqpoint{1.209cm}{1.902cm}}
\pgfpathcurveto{\pgfqpoint{1.172cm}{1.902cm}}{\pgfqpoint{1.138cm}{1.887cm}}{\pgfqpoint{1.112cm}{1.862cm}}
\pgfpathcurveto{\pgfqpoint{1.087cm}{1.836cm}}{\pgfqpoint{1.072cm}{1.801cm}}{\pgfqpoint{1.072cm}{1.765cm}}
\pgfpathcurveto{\pgfqpoint{1.072cm}{1.728cm}}{\pgfqpoint{1.087cm}{1.694cm}}{\pgfqpoint{1.112cm}{1.668cm}}
\pgfpathcurveto{\pgfqpoint{1.138cm}{1.642cm}}{\pgfqpoint{1.172cm}{1.628cm}}{\pgfqpoint{1.209cm}{1.628cm}}
\pgfpathcurveto{\pgfqpoint{1.245cm}{1.628cm}}{\pgfqpoint{1.28cm}{1.642cm}}{\pgfqpoint{1.305cm}{1.668cm}}
\pgfpathcurveto{\pgfqpoint{1.331cm}{1.694cm}}{\pgfqpoint{1.345cm}{1.728cm}}{\pgfqpoint{1.345cm}{1.765cm}}
\pgfusepath{fill}
\begin{pgfscope}
\pgfsetdash{}{0cm}
\pgfsetlinewidth{0.818mm}
\pgfsetroundcap
\pgfsetroundjoin
\pgfsetmiterlimit{7.0}
\pgfpathmoveto{\pgfqpoint{0.682cm}{1.065cm}}
\pgfpathlineto{\pgfqpoint{1.246cm}{0.315cm}}
\pgfpathlineto{\pgfqpoint{1.811cm}{1.065cm}}
\pgfusepath{stroke}
\end{pgfscope}
\pgfpathmoveto{\pgfqpoint{1.948cm}{1.065cm}}
\pgfpathcurveto{\pgfqpoint{1.948cm}{1.101cm}}{\pgfqpoint{1.933cm}{1.136cm}}{\pgfqpoint{1.907cm}{1.162cm}}
\pgfpathcurveto{\pgfqpoint{1.882cm}{1.187cm}}{\pgfqpoint{1.847cm}{1.202cm}}{\pgfqpoint{1.811cm}{1.202cm}}
\pgfpathcurveto{\pgfqpoint{1.775cm}{1.202cm}}{\pgfqpoint{1.74cm}{1.187cm}}{\pgfqpoint{1.714cm}{1.162cm}}
\pgfpathcurveto{\pgfqpoint{1.689cm}{1.136cm}}{\pgfqpoint{1.674cm}{1.101cm}}{\pgfqpoint{1.674cm}{1.065cm}}
\pgfpathcurveto{\pgfqpoint{1.674cm}{1.029cm}}{\pgfqpoint{1.689cm}{0.994cm}}{\pgfqpoint{1.714cm}{0.968cm}}
\pgfpathcurveto{\pgfqpoint{1.74cm}{0.942cm}}{\pgfqpoint{1.775cm}{0.928cm}}{\pgfqpoint{1.811cm}{0.928cm}}
\pgfpathcurveto{\pgfqpoint{1.847cm}{0.928cm}}{\pgfqpoint{1.882cm}{0.942cm}}{\pgfqpoint{1.907cm}{0.968cm}}
\pgfpathcurveto{\pgfqpoint{1.933cm}{0.994cm}}{\pgfqpoint{1.948cm}{1.029cm}}{\pgfqpoint{1.948cm}{1.065cm}}
\pgfusepath{fill}
\begin{pgfscope}
\pgfsetdash{}{0cm}
\pgfsetlinewidth{0.818mm}
\pgfsetmiterlimit{7.0}
\pgfpathmoveto{\pgfqpoint{1.246cm}{0.315cm}}
\pgfpathlineto{\pgfqpoint{1.244cm}{1.061cm}}
\pgfusepath{stroke}
\end{pgfscope}
\pgfpathmoveto{\pgfqpoint{1.38cm}{1.065cm}}
\pgfpathcurveto{\pgfqpoint{1.38cm}{1.101cm}}{\pgfqpoint{1.366cm}{1.136cm}}{\pgfqpoint{1.34cm}{1.162cm}}
\pgfpathcurveto{\pgfqpoint{1.315cm}{1.187cm}}{\pgfqpoint{1.28cm}{1.202cm}}{\pgfqpoint{1.244cm}{1.202cm}}
\pgfpathcurveto{\pgfqpoint{1.207cm}{1.202cm}}{\pgfqpoint{1.173cm}{1.187cm}}{\pgfqpoint{1.147cm}{1.162cm}}
\pgfpathcurveto{\pgfqpoint{1.121cm}{1.136cm}}{\pgfqpoint{1.107cm}{1.101cm}}{\pgfqpoint{1.107cm}{1.065cm}}
\pgfpathcurveto{\pgfqpoint{1.107cm}{1.029cm}}{\pgfqpoint{1.121cm}{0.994cm}}{\pgfqpoint{1.147cm}{0.968cm}}
\pgfpathcurveto{\pgfqpoint{1.173cm}{0.942cm}}{\pgfqpoint{1.207cm}{0.928cm}}{\pgfqpoint{1.244cm}{0.928cm}}
\pgfpathcurveto{\pgfqpoint{1.28cm}{0.928cm}}{\pgfqpoint{1.315cm}{0.942cm}}{\pgfqpoint{1.34cm}{0.968cm}}
\pgfpathcurveto{\pgfqpoint{1.366cm}{0.994cm}}{\pgfqpoint{1.38cm}{1.029cm}}{\pgfqpoint{1.38cm}{1.065cm}}
\pgfusepath{fill}
\begin{pgfscope}
\pgfsetdash{}{0cm}
\pgfsetlinewidth{0.818mm}
\pgfsetmiterlimit{4.0}
\pgfpathmoveto{\pgfqpoint{1.383cm}{0.178cm}}
\pgfpathcurveto{\pgfqpoint{1.383cm}{0.214cm}}{\pgfqpoint{1.369cm}{0.249cm}}{\pgfqpoint{1.343cm}{0.275cm}}
\pgfpathcurveto{\pgfqpoint{1.317cm}{0.3cm}}{\pgfqpoint{1.283cm}{0.315cm}}{\pgfqpoint{1.246cm}{0.315cm}}
\pgfpathcurveto{\pgfqpoint{1.21cm}{0.315cm}}{\pgfqpoint{1.175cm}{0.3cm}}{\pgfqpoint{1.15cm}{0.275cm}}
\pgfpathcurveto{\pgfqpoint{1.124cm}{0.249cm}}{\pgfqpoint{1.11cm}{0.214cm}}{\pgfqpoint{1.11cm}{0.178cm}}
\pgfpathcurveto{\pgfqpoint{1.11cm}{0.141cm}}{\pgfqpoint{1.124cm}{0.107cm}}{\pgfqpoint{1.15cm}{0.081cm}}
\pgfpathcurveto{\pgfqpoint{1.175cm}{0.055cm}}{\pgfqpoint{1.21cm}{0.041cm}}{\pgfqpoint{1.246cm}{0.041cm}}
\pgfpathcurveto{\pgfqpoint{1.283cm}{0.041cm}}{\pgfqpoint{1.317cm}{0.055cm}}{\pgfqpoint{1.343cm}{0.081cm}}
\pgfpathcurveto{\pgfqpoint{1.369cm}{0.107cm}}{\pgfqpoint{1.383cm}{0.141cm}}{\pgfqpoint{1.383cm}{0.178cm}}
\pgfusepath{stroke}
\end{pgfscope}
\end{pgfscope}
\end{pgfscope}
\end{pgfscope}
\end{tikzpicture}}} \|_{\mathscr{C} \hspace{.1em}^{-
     \kappa, \varepsilon}} \| \rho^{\sigma} Y_{\varepsilon} \|_{L^{\infty,
     \varepsilon}} \leqslant \lambda^{(8 - \theta) / (2 + \theta)}
     C_{\delta}^{} \| \mathbb{X}_{\varepsilon} \|^{8 + \vartheta} + \delta
     \Upsilon_{\varepsilon} \]
  \[ \leqslant (\lambda^2 + \lambda^3) C_{\delta}^{} \|
     \mathbb{X}_{\varepsilon} \|^{8 + \vartheta} + \delta
     \Upsilon_{\varepsilon} . \]
  and the resonant term is bounded as
  \[ \lambda^2 | \langle \rho^4 \phi_{\varepsilon},
     \tilde{X}_{\varepsilon}^{\!\resizebox{!}{.8em}{
\begin{tikzpicture}
\pgfpathmoveto{\pgfqpoint{0cm}{-0.035cm}}
\pgfpathlineto{\pgfqpoint{1.976cm}{-0.035cm}}
\pgfpathlineto{\pgfqpoint{1.976cm}{1.94cm}}
\pgfpathlineto{\pgfqpoint{0cm}{1.94cm}}
\pgfpathclose
\pgfusepath{clip}
\begin{pgfscope}
\begin{pgfscope}
\pgfpathmoveto{\pgfqpoint{0cm}{-0.035cm}}
\pgfpathlineto{\pgfqpoint{1.976cm}{-0.035cm}}
\pgfpathlineto{\pgfqpoint{1.976cm}{1.94cm}}
\pgfpathlineto{\pgfqpoint{0cm}{1.94cm}}
\pgfpathclose
\pgfusepath{clip}
\begin{pgfscope}
\begin{pgfscope}
\pgfsetdash{}{0cm}
\pgfsetlinewidth{0.818mm}
\pgfsetroundcap
\pgfsetroundjoin
\pgfsetmiterlimit{7.0}
\definecolor{eps2pgf_color}{gray}{0}\pgfsetstrokecolor{eps2pgf_color}\pgfsetfillcolor{eps2pgf_color}
\pgfpathmoveto{\pgfqpoint{0.117cm}{1.815cm}}
\pgfpathlineto{\pgfqpoint{0.682cm}{1.065cm}}
\pgfpathlineto{\pgfqpoint{1.246cm}{1.815cm}}
\pgfusepath{stroke}
\end{pgfscope}
\definecolor{eps2pgf_color}{gray}{0}\pgfsetstrokecolor{eps2pgf_color}\pgfsetfillcolor{eps2pgf_color}
\pgfpathmoveto{\pgfqpoint{0.273cm}{1.789cm}}
\pgfpathcurveto{\pgfqpoint{0.273cm}{1.825cm}}{\pgfqpoint{0.259cm}{1.86cm}}{\pgfqpoint{0.233cm}{1.886cm}}
\pgfpathcurveto{\pgfqpoint{0.207cm}{1.912cm}}{\pgfqpoint{0.173cm}{1.926cm}}{\pgfqpoint{0.137cm}{1.926cm}}
\pgfpathcurveto{\pgfqpoint{0.1cm}{1.926cm}}{\pgfqpoint{0.066cm}{1.912cm}}{\pgfqpoint{0.04cm}{1.886cm}}
\pgfpathcurveto{\pgfqpoint{0.014cm}{1.86cm}}{\pgfqpoint{0cm}{1.825cm}}{\pgfqpoint{0cm}{1.789cm}}
\pgfpathcurveto{\pgfqpoint{0cm}{1.753cm}}{\pgfqpoint{0.014cm}{1.718cm}}{\pgfqpoint{0.04cm}{1.692cm}}
\pgfpathcurveto{\pgfqpoint{0.066cm}{1.667cm}}{\pgfqpoint{0.1cm}{1.652cm}}{\pgfqpoint{0.137cm}{1.652cm}}
\pgfpathcurveto{\pgfqpoint{0.173cm}{1.652cm}}{\pgfqpoint{0.207cm}{1.667cm}}{\pgfqpoint{0.233cm}{1.692cm}}
\pgfpathcurveto{\pgfqpoint{0.259cm}{1.718cm}}{\pgfqpoint{0.273cm}{1.753cm}}{\pgfqpoint{0.273cm}{1.789cm}}
\pgfusepath{fill}
\pgfpathmoveto{\pgfqpoint{1.345cm}{1.765cm}}
\pgfpathcurveto{\pgfqpoint{1.345cm}{1.801cm}}{\pgfqpoint{1.331cm}{1.836cm}}{\pgfqpoint{1.305cm}{1.862cm}}
\pgfpathcurveto{\pgfqpoint{1.28cm}{1.887cm}}{\pgfqpoint{1.245cm}{1.902cm}}{\pgfqpoint{1.209cm}{1.902cm}}
\pgfpathcurveto{\pgfqpoint{1.172cm}{1.902cm}}{\pgfqpoint{1.138cm}{1.887cm}}{\pgfqpoint{1.112cm}{1.862cm}}
\pgfpathcurveto{\pgfqpoint{1.087cm}{1.836cm}}{\pgfqpoint{1.072cm}{1.801cm}}{\pgfqpoint{1.072cm}{1.765cm}}
\pgfpathcurveto{\pgfqpoint{1.072cm}{1.728cm}}{\pgfqpoint{1.087cm}{1.694cm}}{\pgfqpoint{1.112cm}{1.668cm}}
\pgfpathcurveto{\pgfqpoint{1.138cm}{1.642cm}}{\pgfqpoint{1.172cm}{1.628cm}}{\pgfqpoint{1.209cm}{1.628cm}}
\pgfpathcurveto{\pgfqpoint{1.245cm}{1.628cm}}{\pgfqpoint{1.28cm}{1.642cm}}{\pgfqpoint{1.305cm}{1.668cm}}
\pgfpathcurveto{\pgfqpoint{1.331cm}{1.694cm}}{\pgfqpoint{1.345cm}{1.728cm}}{\pgfqpoint{1.345cm}{1.765cm}}
\pgfusepath{fill}
\begin{pgfscope}
\pgfsetdash{}{0cm}
\pgfsetlinewidth{0.818mm}
\pgfsetroundcap
\pgfsetroundjoin
\pgfsetmiterlimit{7.0}
\pgfpathmoveto{\pgfqpoint{0.682cm}{1.065cm}}
\pgfpathlineto{\pgfqpoint{1.246cm}{0.315cm}}
\pgfpathlineto{\pgfqpoint{1.811cm}{1.065cm}}
\pgfusepath{stroke}
\end{pgfscope}
\pgfpathmoveto{\pgfqpoint{1.948cm}{1.065cm}}
\pgfpathcurveto{\pgfqpoint{1.948cm}{1.101cm}}{\pgfqpoint{1.933cm}{1.136cm}}{\pgfqpoint{1.907cm}{1.162cm}}
\pgfpathcurveto{\pgfqpoint{1.882cm}{1.187cm}}{\pgfqpoint{1.847cm}{1.202cm}}{\pgfqpoint{1.811cm}{1.202cm}}
\pgfpathcurveto{\pgfqpoint{1.775cm}{1.202cm}}{\pgfqpoint{1.74cm}{1.187cm}}{\pgfqpoint{1.714cm}{1.162cm}}
\pgfpathcurveto{\pgfqpoint{1.689cm}{1.136cm}}{\pgfqpoint{1.674cm}{1.101cm}}{\pgfqpoint{1.674cm}{1.065cm}}
\pgfpathcurveto{\pgfqpoint{1.674cm}{1.029cm}}{\pgfqpoint{1.689cm}{0.994cm}}{\pgfqpoint{1.714cm}{0.968cm}}
\pgfpathcurveto{\pgfqpoint{1.74cm}{0.942cm}}{\pgfqpoint{1.775cm}{0.928cm}}{\pgfqpoint{1.811cm}{0.928cm}}
\pgfpathcurveto{\pgfqpoint{1.847cm}{0.928cm}}{\pgfqpoint{1.882cm}{0.942cm}}{\pgfqpoint{1.907cm}{0.968cm}}
\pgfpathcurveto{\pgfqpoint{1.933cm}{0.994cm}}{\pgfqpoint{1.948cm}{1.029cm}}{\pgfqpoint{1.948cm}{1.065cm}}
\pgfusepath{fill}
\begin{pgfscope}
\pgfsetdash{}{0cm}
\pgfsetlinewidth{0.818mm}
\pgfsetmiterlimit{7.0}
\pgfpathmoveto{\pgfqpoint{1.246cm}{0.315cm}}
\pgfpathlineto{\pgfqpoint{1.244cm}{1.061cm}}
\pgfusepath{stroke}
\end{pgfscope}
\pgfpathmoveto{\pgfqpoint{1.38cm}{1.065cm}}
\pgfpathcurveto{\pgfqpoint{1.38cm}{1.101cm}}{\pgfqpoint{1.366cm}{1.136cm}}{\pgfqpoint{1.34cm}{1.162cm}}
\pgfpathcurveto{\pgfqpoint{1.315cm}{1.187cm}}{\pgfqpoint{1.28cm}{1.202cm}}{\pgfqpoint{1.244cm}{1.202cm}}
\pgfpathcurveto{\pgfqpoint{1.207cm}{1.202cm}}{\pgfqpoint{1.173cm}{1.187cm}}{\pgfqpoint{1.147cm}{1.162cm}}
\pgfpathcurveto{\pgfqpoint{1.121cm}{1.136cm}}{\pgfqpoint{1.107cm}{1.101cm}}{\pgfqpoint{1.107cm}{1.065cm}}
\pgfpathcurveto{\pgfqpoint{1.107cm}{1.029cm}}{\pgfqpoint{1.121cm}{0.994cm}}{\pgfqpoint{1.147cm}{0.968cm}}
\pgfpathcurveto{\pgfqpoint{1.173cm}{0.942cm}}{\pgfqpoint{1.207cm}{0.928cm}}{\pgfqpoint{1.244cm}{0.928cm}}
\pgfpathcurveto{\pgfqpoint{1.28cm}{0.928cm}}{\pgfqpoint{1.315cm}{0.942cm}}{\pgfqpoint{1.34cm}{0.968cm}}
\pgfpathcurveto{\pgfqpoint{1.366cm}{0.994cm}}{\pgfqpoint{1.38cm}{1.029cm}}{\pgfqpoint{1.38cm}{1.065cm}}
\pgfusepath{fill}
\begin{pgfscope}
\pgfsetdash{}{0cm}
\pgfsetlinewidth{0.818mm}
\pgfsetmiterlimit{4.0}
\pgfpathmoveto{\pgfqpoint{1.383cm}{0.178cm}}
\pgfpathcurveto{\pgfqpoint{1.383cm}{0.214cm}}{\pgfqpoint{1.369cm}{0.249cm}}{\pgfqpoint{1.343cm}{0.275cm}}
\pgfpathcurveto{\pgfqpoint{1.317cm}{0.3cm}}{\pgfqpoint{1.283cm}{0.315cm}}{\pgfqpoint{1.246cm}{0.315cm}}
\pgfpathcurveto{\pgfqpoint{1.21cm}{0.315cm}}{\pgfqpoint{1.175cm}{0.3cm}}{\pgfqpoint{1.15cm}{0.275cm}}
\pgfpathcurveto{\pgfqpoint{1.124cm}{0.249cm}}{\pgfqpoint{1.11cm}{0.214cm}}{\pgfqpoint{1.11cm}{0.178cm}}
\pgfpathcurveto{\pgfqpoint{1.11cm}{0.141cm}}{\pgfqpoint{1.124cm}{0.107cm}}{\pgfqpoint{1.15cm}{0.081cm}}
\pgfpathcurveto{\pgfqpoint{1.175cm}{0.055cm}}{\pgfqpoint{1.21cm}{0.041cm}}{\pgfqpoint{1.246cm}{0.041cm}}
\pgfpathcurveto{\pgfqpoint{1.283cm}{0.041cm}}{\pgfqpoint{1.317cm}{0.055cm}}{\pgfqpoint{1.343cm}{0.081cm}}
\pgfpathcurveto{\pgfqpoint{1.369cm}{0.107cm}}{\pgfqpoint{1.383cm}{0.141cm}}{\pgfqpoint{1.383cm}{0.178cm}}
\pgfusepath{stroke}
\end{pgfscope}
\end{pgfscope}
\end{pgfscope}
\end{pgfscope}
\end{tikzpicture}}} \circ Y_{\varepsilon}
     \rangle_{\varepsilon} | \lesssim \lambda^2 \| \rho^{4 - 2 \sigma}
     \phi_{\varepsilon} \|_{L^{1, \varepsilon}} \| \rho^{\sigma}
     \tilde{X}_{\varepsilon}^{\!\resizebox{!}{.8em}{
\begin{tikzpicture}
\pgfpathmoveto{\pgfqpoint{0cm}{-0.035cm}}
\pgfpathlineto{\pgfqpoint{1.976cm}{-0.035cm}}
\pgfpathlineto{\pgfqpoint{1.976cm}{1.94cm}}
\pgfpathlineto{\pgfqpoint{0cm}{1.94cm}}
\pgfpathclose
\pgfusepath{clip}
\begin{pgfscope}
\begin{pgfscope}
\pgfpathmoveto{\pgfqpoint{0cm}{-0.035cm}}
\pgfpathlineto{\pgfqpoint{1.976cm}{-0.035cm}}
\pgfpathlineto{\pgfqpoint{1.976cm}{1.94cm}}
\pgfpathlineto{\pgfqpoint{0cm}{1.94cm}}
\pgfpathclose
\pgfusepath{clip}
\begin{pgfscope}
\begin{pgfscope}
\pgfsetdash{}{0cm}
\pgfsetlinewidth{0.818mm}
\pgfsetroundcap
\pgfsetroundjoin
\pgfsetmiterlimit{7.0}
\definecolor{eps2pgf_color}{gray}{0}\pgfsetstrokecolor{eps2pgf_color}\pgfsetfillcolor{eps2pgf_color}
\pgfpathmoveto{\pgfqpoint{0.117cm}{1.815cm}}
\pgfpathlineto{\pgfqpoint{0.682cm}{1.065cm}}
\pgfpathlineto{\pgfqpoint{1.246cm}{1.815cm}}
\pgfusepath{stroke}
\end{pgfscope}
\definecolor{eps2pgf_color}{gray}{0}\pgfsetstrokecolor{eps2pgf_color}\pgfsetfillcolor{eps2pgf_color}
\pgfpathmoveto{\pgfqpoint{0.273cm}{1.789cm}}
\pgfpathcurveto{\pgfqpoint{0.273cm}{1.825cm}}{\pgfqpoint{0.259cm}{1.86cm}}{\pgfqpoint{0.233cm}{1.886cm}}
\pgfpathcurveto{\pgfqpoint{0.207cm}{1.912cm}}{\pgfqpoint{0.173cm}{1.926cm}}{\pgfqpoint{0.137cm}{1.926cm}}
\pgfpathcurveto{\pgfqpoint{0.1cm}{1.926cm}}{\pgfqpoint{0.066cm}{1.912cm}}{\pgfqpoint{0.04cm}{1.886cm}}
\pgfpathcurveto{\pgfqpoint{0.014cm}{1.86cm}}{\pgfqpoint{0cm}{1.825cm}}{\pgfqpoint{0cm}{1.789cm}}
\pgfpathcurveto{\pgfqpoint{0cm}{1.753cm}}{\pgfqpoint{0.014cm}{1.718cm}}{\pgfqpoint{0.04cm}{1.692cm}}
\pgfpathcurveto{\pgfqpoint{0.066cm}{1.667cm}}{\pgfqpoint{0.1cm}{1.652cm}}{\pgfqpoint{0.137cm}{1.652cm}}
\pgfpathcurveto{\pgfqpoint{0.173cm}{1.652cm}}{\pgfqpoint{0.207cm}{1.667cm}}{\pgfqpoint{0.233cm}{1.692cm}}
\pgfpathcurveto{\pgfqpoint{0.259cm}{1.718cm}}{\pgfqpoint{0.273cm}{1.753cm}}{\pgfqpoint{0.273cm}{1.789cm}}
\pgfusepath{fill}
\pgfpathmoveto{\pgfqpoint{1.345cm}{1.765cm}}
\pgfpathcurveto{\pgfqpoint{1.345cm}{1.801cm}}{\pgfqpoint{1.331cm}{1.836cm}}{\pgfqpoint{1.305cm}{1.862cm}}
\pgfpathcurveto{\pgfqpoint{1.28cm}{1.887cm}}{\pgfqpoint{1.245cm}{1.902cm}}{\pgfqpoint{1.209cm}{1.902cm}}
\pgfpathcurveto{\pgfqpoint{1.172cm}{1.902cm}}{\pgfqpoint{1.138cm}{1.887cm}}{\pgfqpoint{1.112cm}{1.862cm}}
\pgfpathcurveto{\pgfqpoint{1.087cm}{1.836cm}}{\pgfqpoint{1.072cm}{1.801cm}}{\pgfqpoint{1.072cm}{1.765cm}}
\pgfpathcurveto{\pgfqpoint{1.072cm}{1.728cm}}{\pgfqpoint{1.087cm}{1.694cm}}{\pgfqpoint{1.112cm}{1.668cm}}
\pgfpathcurveto{\pgfqpoint{1.138cm}{1.642cm}}{\pgfqpoint{1.172cm}{1.628cm}}{\pgfqpoint{1.209cm}{1.628cm}}
\pgfpathcurveto{\pgfqpoint{1.245cm}{1.628cm}}{\pgfqpoint{1.28cm}{1.642cm}}{\pgfqpoint{1.305cm}{1.668cm}}
\pgfpathcurveto{\pgfqpoint{1.331cm}{1.694cm}}{\pgfqpoint{1.345cm}{1.728cm}}{\pgfqpoint{1.345cm}{1.765cm}}
\pgfusepath{fill}
\begin{pgfscope}
\pgfsetdash{}{0cm}
\pgfsetlinewidth{0.818mm}
\pgfsetroundcap
\pgfsetroundjoin
\pgfsetmiterlimit{7.0}
\pgfpathmoveto{\pgfqpoint{0.682cm}{1.065cm}}
\pgfpathlineto{\pgfqpoint{1.246cm}{0.315cm}}
\pgfpathlineto{\pgfqpoint{1.811cm}{1.065cm}}
\pgfusepath{stroke}
\end{pgfscope}
\pgfpathmoveto{\pgfqpoint{1.948cm}{1.065cm}}
\pgfpathcurveto{\pgfqpoint{1.948cm}{1.101cm}}{\pgfqpoint{1.933cm}{1.136cm}}{\pgfqpoint{1.907cm}{1.162cm}}
\pgfpathcurveto{\pgfqpoint{1.882cm}{1.187cm}}{\pgfqpoint{1.847cm}{1.202cm}}{\pgfqpoint{1.811cm}{1.202cm}}
\pgfpathcurveto{\pgfqpoint{1.775cm}{1.202cm}}{\pgfqpoint{1.74cm}{1.187cm}}{\pgfqpoint{1.714cm}{1.162cm}}
\pgfpathcurveto{\pgfqpoint{1.689cm}{1.136cm}}{\pgfqpoint{1.674cm}{1.101cm}}{\pgfqpoint{1.674cm}{1.065cm}}
\pgfpathcurveto{\pgfqpoint{1.674cm}{1.029cm}}{\pgfqpoint{1.689cm}{0.994cm}}{\pgfqpoint{1.714cm}{0.968cm}}
\pgfpathcurveto{\pgfqpoint{1.74cm}{0.942cm}}{\pgfqpoint{1.775cm}{0.928cm}}{\pgfqpoint{1.811cm}{0.928cm}}
\pgfpathcurveto{\pgfqpoint{1.847cm}{0.928cm}}{\pgfqpoint{1.882cm}{0.942cm}}{\pgfqpoint{1.907cm}{0.968cm}}
\pgfpathcurveto{\pgfqpoint{1.933cm}{0.994cm}}{\pgfqpoint{1.948cm}{1.029cm}}{\pgfqpoint{1.948cm}{1.065cm}}
\pgfusepath{fill}
\begin{pgfscope}
\pgfsetdash{}{0cm}
\pgfsetlinewidth{0.818mm}
\pgfsetmiterlimit{7.0}
\pgfpathmoveto{\pgfqpoint{1.246cm}{0.315cm}}
\pgfpathlineto{\pgfqpoint{1.244cm}{1.061cm}}
\pgfusepath{stroke}
\end{pgfscope}
\pgfpathmoveto{\pgfqpoint{1.38cm}{1.065cm}}
\pgfpathcurveto{\pgfqpoint{1.38cm}{1.101cm}}{\pgfqpoint{1.366cm}{1.136cm}}{\pgfqpoint{1.34cm}{1.162cm}}
\pgfpathcurveto{\pgfqpoint{1.315cm}{1.187cm}}{\pgfqpoint{1.28cm}{1.202cm}}{\pgfqpoint{1.244cm}{1.202cm}}
\pgfpathcurveto{\pgfqpoint{1.207cm}{1.202cm}}{\pgfqpoint{1.173cm}{1.187cm}}{\pgfqpoint{1.147cm}{1.162cm}}
\pgfpathcurveto{\pgfqpoint{1.121cm}{1.136cm}}{\pgfqpoint{1.107cm}{1.101cm}}{\pgfqpoint{1.107cm}{1.065cm}}
\pgfpathcurveto{\pgfqpoint{1.107cm}{1.029cm}}{\pgfqpoint{1.121cm}{0.994cm}}{\pgfqpoint{1.147cm}{0.968cm}}
\pgfpathcurveto{\pgfqpoint{1.173cm}{0.942cm}}{\pgfqpoint{1.207cm}{0.928cm}}{\pgfqpoint{1.244cm}{0.928cm}}
\pgfpathcurveto{\pgfqpoint{1.28cm}{0.928cm}}{\pgfqpoint{1.315cm}{0.942cm}}{\pgfqpoint{1.34cm}{0.968cm}}
\pgfpathcurveto{\pgfqpoint{1.366cm}{0.994cm}}{\pgfqpoint{1.38cm}{1.029cm}}{\pgfqpoint{1.38cm}{1.065cm}}
\pgfusepath{fill}
\begin{pgfscope}
\pgfsetdash{}{0cm}
\pgfsetlinewidth{0.818mm}
\pgfsetmiterlimit{4.0}
\pgfpathmoveto{\pgfqpoint{1.383cm}{0.178cm}}
\pgfpathcurveto{\pgfqpoint{1.383cm}{0.214cm}}{\pgfqpoint{1.369cm}{0.249cm}}{\pgfqpoint{1.343cm}{0.275cm}}
\pgfpathcurveto{\pgfqpoint{1.317cm}{0.3cm}}{\pgfqpoint{1.283cm}{0.315cm}}{\pgfqpoint{1.246cm}{0.315cm}}
\pgfpathcurveto{\pgfqpoint{1.21cm}{0.315cm}}{\pgfqpoint{1.175cm}{0.3cm}}{\pgfqpoint{1.15cm}{0.275cm}}
\pgfpathcurveto{\pgfqpoint{1.124cm}{0.249cm}}{\pgfqpoint{1.11cm}{0.214cm}}{\pgfqpoint{1.11cm}{0.178cm}}
\pgfpathcurveto{\pgfqpoint{1.11cm}{0.141cm}}{\pgfqpoint{1.124cm}{0.107cm}}{\pgfqpoint{1.15cm}{0.081cm}}
\pgfpathcurveto{\pgfqpoint{1.175cm}{0.055cm}}{\pgfqpoint{1.21cm}{0.041cm}}{\pgfqpoint{1.246cm}{0.041cm}}
\pgfpathcurveto{\pgfqpoint{1.283cm}{0.041cm}}{\pgfqpoint{1.317cm}{0.055cm}}{\pgfqpoint{1.343cm}{0.081cm}}
\pgfpathcurveto{\pgfqpoint{1.369cm}{0.107cm}}{\pgfqpoint{1.383cm}{0.141cm}}{\pgfqpoint{1.383cm}{0.178cm}}
\pgfusepath{stroke}
\end{pgfscope}
\end{pgfscope}
\end{pgfscope}
\end{pgfscope}
\end{tikzpicture}}} \|_{\mathscr{C} \hspace{.1em}^{-
     \kappa, \varepsilon}} \| \rho^{\sigma} Y_{\varepsilon} \|_{\mathscr{C}
     \hspace{.1em}^{2 \kappa, \varepsilon}} \lesssim \lambda^3 \| \rho
     \phi_{\varepsilon} \|_{L^{4, \varepsilon}} \| \mathbb{X}_{\varepsilon}
     \|^{6 + \vartheta} \]
  \[ \leqslant C_{\delta} \lambda^{11 / 3} \| \mathbb{X}_{\varepsilon} \|^{8 +
     \vartheta} + \delta \Upsilon_{\varepsilon} \leqslant (\lambda^3 +
     \lambda^4) C_{\delta}^{} \| \mathbb{X}_{\varepsilon} \|^{8 + \vartheta} +
     \delta \Upsilon_{\varepsilon} . \]
  Now,
  \[ \lambda^2 | \langle \rho^4 \phi_{\varepsilon}, (\tilde{b}_{\varepsilon} -
     b_{\varepsilon}) Y_{\varepsilon} \rangle_{\varepsilon} | \lesssim | \log
     t | \lambda^2 \| \rho^{4 - \sigma} \phi_{\varepsilon} \|_{L^{1,
     \varepsilon}} \| \rho^{\sigma} Y_{\varepsilon} \|_{L^{\infty,
     \varepsilon}} \lesssim | \log t |^{4 / 3} \lambda^{7 / 3} C_{\delta} \|
     \mathbb{X}_{\varepsilon} \|^{8 / 3} + \delta \Upsilon_{\varepsilon} . \]
  Next, for $\theta = \frac{1 - 5 \kappa}{1 - 2 \kappa}$,
  \[ \lambda^2 | \langle \rho^4 \phi_{\varepsilon}, \bar{C}_{\varepsilon}
     (Y_{\varepsilon}, 3 \llbracket X_{\varepsilon}^2 \rrbracket, 3 \llbracket
     X_{\varepsilon}^2 \rrbracket) \rangle_{\varepsilon} | \lesssim \lambda^2
     \| \rho^{4 - 3 \sigma} \phi_{\varepsilon} \|_{B^{2 \kappa,
     \varepsilon}_{1, 1}} \| \rho^{\sigma} Y_{\varepsilon} \|_{\mathscr{C}
     \hspace{.1em}^{2 \kappa, \varepsilon}} \| \rho^{\sigma} \llbracket
     X_{\varepsilon}^2 \rrbracket \|^2_{\mathscr{C} \hspace{.1em}^{- 1 -
     \kappa, \varepsilon}} \]
  \[ \lesssim \lambda^3 \| \rho \phi_{\varepsilon} \|_{L^{4,
     \varepsilon}}^{\theta} \| \rho^2 \phi_{\varepsilon} \|_{H^{1 - 2 \kappa,
     \varepsilon}}^{1 - \theta} \| \mathbb{X}_{\varepsilon} \|^{6 + \vartheta}
     \leqslant \lambda^{(12 - \theta) / (2 + \theta)} C_{\delta} \|
     \mathbb{X}_{\varepsilon} \|^{8 + \vartheta} + \delta
     \Upsilon_{\varepsilon} \]
  \[ \leqslant (\lambda^3 + \lambda^4) C_{\delta}^{} \|
     \mathbb{X}_{\varepsilon} \|^{8 + \vartheta} + \delta
     \Upsilon_{\varepsilon} . \]
  At last, we have
  \[ \lambda^2 \left| \left\langle \rho^4 \phi_{\varepsilon}, - 3 \llbracket
     X_{\varepsilon}^2 \rrbracket \circ \LL_{\varepsilon}^{- 1} \left( 3
     \UU^{\varepsilon}_{\leqslant} \llbracket X_{\varepsilon}^2 \rrbracket
     \succ Y_{\varepsilon} \right) \right\rangle_{\varepsilon} \right| \]
  \[ \lesssim \lambda^2 \| \rho^{4 - 3 \sigma} \phi_{\varepsilon} \|_{L^{1,
     \varepsilon}} \| \rho^{\sigma} Y_{\varepsilon} \|_{L^{\infty,
     \varepsilon}} \| \rho^{\sigma} \llbracket X_{\varepsilon}^2 \rrbracket
     \|_{\mathscr{C} \hspace{.1em}^{- 1 - \kappa, \varepsilon}} \left\|
     \rho^{\sigma} \UU^{\varepsilon}_{\leqslant} \llbracket X_{\varepsilon}^2
     \rrbracket \right\|_{\mathscr{C} \hspace{.1em}^{- 1 + 2 \kappa,
     \varepsilon}} \]
  \[ \lesssim \lambda^3 \| \rho^{4 - 3 \sigma} \phi_{\varepsilon} \|_{L^{1,
     \varepsilon}} \| \mathbb{X}_{\varepsilon} \|^{4 + \vartheta} \leqslant
     \lambda^{11 / 3} C_{\delta} \| \mathbb{X}_{\varepsilon} \|^{16 / 3 +
     \vartheta} + \delta \Upsilon_{\varepsilon} \leqslant (\lambda^3 +
     \lambda^4) C_{\delta}^{} \| \mathbb{X}_{\varepsilon} \|^{8 + \vartheta} +
     \delta \Upsilon_{\varepsilon} \]
  This concludes the estimation of $\langle \rho^4 \phi_{\varepsilon},
  \lambda^2 Z_{\varepsilon} \rangle_{\varepsilon}$ giving us
  \[ | \langle \rho^4 \phi_{\varepsilon}, \lambda^2 Z_{\varepsilon}
     \rangle_{\varepsilon} | \leqslant (\lambda^2 + \lambda^4) C_{\delta}^{}
     \| \mathbb{X}_{\varepsilon} \|^{8 + \vartheta} + \delta
     \Upsilon_{\varepsilon} . \]
  Finally, we arrive to the additional term introduced by the localization. Using \eqref{eq:bound-X3} we obtain
  \[ | \langle \rho^4 \phi_{\varepsilon}, - \lambda \llbracket X_{M,
     \varepsilon}^3 \rrbracket_{\leqslant} \rangle_{\varepsilon} | \lesssim
     \lambda \| \rho \phi_{\varepsilon} \|_{L^{4, \varepsilon}} \|
     \rho^{\sigma} \llbracket X_{M, \varepsilon}^3 \rrbracket_{\leqslant}
     \|_{C_T L^{\infty, \varepsilon}} \lesssim \lambda \| \rho
     \phi_{\varepsilon} \|_{L^{4, \varepsilon}} 2^{K (3 / 2 + \kappa)} \|
     \mathbb{X}_{\varepsilon} \|^3 \]
  \[ \leqslant \lambda C_{\delta}^{} \| \mathbb{X}_{\varepsilon} \|^{8 +
     \vartheta} + \delta \Upsilon_{\varepsilon}, \]
  where we also see that the power $8+\vartheta$ is optimal for this decomposition.
\end{proof}

Let $\langle \phi_{\varepsilon} \rangle \assign (1 + \| \rho^2
\phi_{\varepsilon} \|_{L^{2, \varepsilon}}^2)^{1 / 2}$ and $\langle
\varphi_{\varepsilon} \rangle_{\ast} \assign (1 + \| \rho^2
\varphi_{\varepsilon} \|_{H^{- 1 / 2 - 2 \kappa, \varepsilon}}^2)^{1 / 2}$.
With Lemma~\ref{lemma:bounds-rhs1-int} in hand we can proceed to the proof of the stretched exponential
integrability.

\begin{proposition}
  \label{lemma:int-bound}
  There exists an $\alpha
  > 0$, $0 < C < 1$ and  $\upsilon=O(\kappa)>0$  such that for every $\beta>0$ 
  \[ \partial_t e^{\beta \langle t \phi_{\varepsilon} \rangle^{1 - \upsilon}}
     + \alpha e^{\beta \langle t \phi_{\varepsilon} \rangle^{1 - \upsilon}} (1-\upsilon)\beta
     \langle t \phi_{\varepsilon} \rangle^{- \upsilon - 1} t^2
     \Upsilon_{\varepsilon} \lesssim 1 + e^{(\beta / C) \|
     \mathbb{X}_{\varepsilon} \|^2}. \]
 Consequently, for any accumulation point $\nu$ we have
  \[ \int_{\mathcal{S}'(\mathbb{R}^{3})} e^{\beta \langle \varphi_{} \rangle_{\ast}^{1 - \upsilon}} \nu
     (\mathd \varphi) < \infty \]
provided  $\beta>0$ is sufficiently small.
\end{proposition}

\begin{proof}
We apply \eqref{eq:en12-int} and Lemma~\ref{lemma:bounds-rhs1-int} to obtain
  \[ \langle t \phi_{\varepsilon} \rangle^{1+ \upsilon} \frac{\partial_t e^{\beta \langle t \phi_{\varepsilon} \rangle^{1 - \upsilon}}}{(1 -
     \upsilon) \beta}
     = e^{\beta \langle t \phi_{\varepsilon} \rangle^{1 - \upsilon}} 
     \frac{1}{2} \partial_t (t^2 \| \rho^2 \phi_{\varepsilon} \|_{L^{2,
     \varepsilon}}^2) \]
  \[ = e^{\beta \langle t \phi_{\varepsilon} \rangle^{1 - \upsilon}}
     [t^2 (- \Upsilon_{\varepsilon} + \Theta_{\rho^4, \varepsilon} +
     \Psi_{\rho^4, \varepsilon}+\langle \rho^4 \phi_{\varepsilon}, - \lambda \llbracket X_{
     \varepsilon}^3 \rrbracket_{\leqslant} \rangle_{\varepsilon}) + t \| \rho^2 \phi_{\varepsilon} \|_{L^{2,
     \varepsilon}}^2] \]
  \[ \leqslant e^{\beta \langle t \phi_{\varepsilon} \rangle^{1 - \upsilon}}
     [t^2 (- \Upsilon_{\varepsilon} + \Theta_{\rho^4, \varepsilon} +
     \Psi_{\rho^4, \varepsilon}+\langle \rho^4 \phi_{\varepsilon}, - \lambda \llbracket X_{
     \varepsilon}^3 \rrbracket_{\leqslant} \rangle_{\varepsilon}) + \delta t^2\lambda \| \rho \phi_{\varepsilon}
     \|_{L^{4, \varepsilon}}^4 + C_{\delta,\lambda^{-1}}] \]
  \[ \leqslant e^{\beta \langle t \phi_{\varepsilon} \rangle^{1 - \upsilon}}
      [- t^2 (1 - 2 \delta) \Upsilon_{\varepsilon} + C_{\lambda} t^2 (| \log
     t|^{4 / 3} + 1) \| \mathbb{X}_{\varepsilon} \|^{8 + \vartheta} +
     C_{\delta,\lambda^{-1}}], \]
     where by writing $C_{\delta,\lambda^{-1}}$ we point out that the constant is not uniform over small $\lambda$.
  Therefore by absorbing the constant term $C_{\delta,\lambda^{-1}}$ in $\|
  \mathbb{X}_{\varepsilon} \|^{8 + \vartheta}$ we have
  \begin{equation}
    \begin{array}{l}
      \partial_t e^{\beta \langle t \phi_{\varepsilon} \rangle^{1 - \upsilon}}
      + e^{\beta \langle t \phi_{\varepsilon} \rangle^{1 - \upsilon}} (1 -
      \upsilon) \beta \langle t \phi_{\varepsilon} \rangle^{- \upsilon - 1} (1
      - 2 \delta) t^2 \Upsilon_{\varepsilon}\\
      \qquad \leqslant C_{\delta,\lambda^{-1}} e^{\beta \langle t
      \phi_{\varepsilon} \rangle^{1 - \upsilon}} (1 - \upsilon) \beta \langle
      t \phi_{\varepsilon} \rangle^{- \upsilon - 1} t^2 (| \log t|^{4 / 3} +
      1) \| \mathbb{X}_{\varepsilon} \|^{8 + \vartheta}
    \end{array} \label{eq:int-bound}
  \end{equation}
  Now we can have two situations at any given time, either $\|
  \mathbb{X}_{\varepsilon} \|^2 \leqslant \varsigma \|t \rho
  \phi_{\varepsilon} \|_{L^{4, \varepsilon}}^{1 - \upsilon}$ or $\|
  \mathbb{X}_{\varepsilon} \|^2 > \varsigma \|t \rho \phi_{\varepsilon}
  \|_{L^{4, \varepsilon}}^{1 - \upsilon}$ for some fixed and small $\varsigma
  > 0$. In the first case the right hand side of~{\eqref{eq:int-bound}} is bounded by
  \[ C_{\delta,\lambda^{-1}} e^{\beta \langle t \phi_{\varepsilon} \rangle^{1 -
     \upsilon}} (1 - \upsilon)  \beta \langle t \phi_{\varepsilon} \rangle^{-
     \upsilon - 1} \varsigma^{4 + \vartheta / 2} t^2 (| \log t|^{4 / 3} + 1)
     \|t \rho \phi_{\varepsilon} \|_{L^{4, \varepsilon}}^{(4 + \vartheta / 2)
     (1 - \upsilon)}, \]
  and we can choose $\upsilon = \upsilon (\kappa)$ so that $(4 + \vartheta /
  2) (1 - \upsilon) = 4$ and by taking $\varsigma$ small (depending on $\delta, \lambda$ through $C_{\delta,\lambda^{-1}}$) we can absorb this
  term into the left hand side since for $t \in (0, 1)$ it will be bounded by
  \[ C_{\delta,\lambda^{-1}} e^{\beta \langle t \phi_{\varepsilon} \rangle^{1 - \upsilon}}(1 - \upsilon)  \beta \langle t \phi_{\varepsilon} \rangle^{-
     \upsilon - 1}
     \varsigma^{4 + \vartheta / 2} \tmcolor{red}{} t^2 \| \rho
     \phi_{\varepsilon} \|_{L^{4, \varepsilon}}^4 . \]
  In the case $\| \mathbb{X}_{\varepsilon} \|^2 > \varsigma \|t \rho
  \phi_{\varepsilon} \|_{L^{4, \varepsilon}}^{1 - \upsilon}$ we have
  \[ \| \mathbb{X}_{\varepsilon} \|^2 > \varsigma \|t \rho \phi_{\varepsilon}
     \|_{L^{4, \varepsilon}}^{1 - \upsilon} \gtrsim \varsigma \|t \rho^2
     \phi_{\varepsilon} \|_{L^{2, \varepsilon}}^{1 - \upsilon} \gtrsim
     \varsigma (\langle t \phi_{\varepsilon} \rangle^{1 - \upsilon} - 1), \]
  provided $\rho$ is chosen to be of sufficient decay, and therefore we simply
  bound the right hand side of {\eqref{eq:int-bound}} by
  \[ \lesssim C_{\delta,\lambda^{-1}} e^{(\beta / C \varsigma) \|
     \mathbb{X}_{\varepsilon} \|^2} \| \mathbb{X}_{\varepsilon} \|^{8 +
     \vartheta} \lesssim 1 + e^{(2 \beta / C \varsigma) \|
     \mathbb{X}_{\varepsilon} \|^2} . \]
    The first claim is proven.
     
  It remains to prove the bound for $\varphi_{\varepsilon}$. By H{\"o}lder's inequality, we have
  \[ \mathbb{E} [e^{\beta \langle \varphi_{\varepsilon} (0) - X_{\varepsilon}
     (0) \rangle^{1 - \upsilon}}] =\mathbb{E} [e^{\beta \langle
     \varphi_{\varepsilon} (1) - X_{\varepsilon} (1) \rangle^{1 - \upsilon}}]
     \leqslant \mathbb{E} [e^{ \beta \tmcolor{red}{} \langle Y_{\varepsilon}
     (1) \rangle^{1 - \upsilon} +  \beta \tmcolor{red}{} \langle
     \phi_{\varepsilon} (1) \rangle^{1 - \upsilon}}] \]
  \[ \leqslant [\mathbb{E} [e^{2  \beta \langle Y_{\varepsilon} (1)
     \rangle^{1 - \upsilon}}]]^{1 / 2} [\mathbb{E} [e^{2 \beta  \langle
     \phi_{\varepsilon} (1) \rangle^{1 - \upsilon}}]]^{1 / 2} \]
  and we observe that $\langle Y_{\varepsilon} (1) \rangle^{1 - \upsilon}
  \lesssim 1 + \| \mathbb{X}_{\varepsilon} \|^2$ so the first term on the
  right hand side is integrable uniformly in $\varepsilon$ by {\eqref{eq:exp-int}}. On
  the other hand, using Lemma~\ref{lemma:int-bound} we have
  \[ \mathbb{E} [e^{2  \beta \langle t \phi_{\varepsilon} (t) \rangle^{1 -
     \upsilon}}] + \int_0^t \mathbb{E} [\alpha e^{2  \beta \langle s
     \phi_{\varepsilon} (s) \rangle^{1 - \upsilon}} (1-\upsilon)2\beta \langle s
     \phi_{\varepsilon} (s) \rangle^{- \upsilon - 1} s^2
     \Upsilon_{\varepsilon} (s)] \mathd s \lesssim \mathbb{E} [1 + e^{(2 \beta
     / C) \| \mathbb{X}_{\varepsilon} \|^2}] \]
  and therefore
  \[ \mathbb{E} [e^{2  \beta \langle \phi_{\varepsilon} (1) \rangle^{1 -
     \upsilon}}] \lesssim \mathbb{E} [1 + e^{(2 \beta / C) \|
     \mathbb{X}_{\varepsilon} \|^2}] . \]
  We conclude that
  \[ \sup_{\varepsilon\in\mathcal A}\mathbb{E} [e^{\beta \langle \varphi_{\varepsilon} (0) - X_{\varepsilon}
     (0) \rangle^{1 - \upsilon}}] \lesssim [\mathbb{E} [e^{2 \beta  (1+\|\mathbb X_{\varepsilon}\|^{2})}]]^{1 / 2} [\mathbb{E} [1 +
     e^{(2 \beta / C) \| \mathbb{X}_{\varepsilon} \|^2}]]^{1 / 2} < \infty \]
  uniformly in $\varepsilon$ by {\eqref{eq:exp-int}}, from which the claim
   follows.
\end{proof}

\section{The Osterwalder--Schrader axioms and non-Gaussianity}\label{s:ax}

The goal of this section is to establish several important properties of any
limit measure $\nu$ obtained in the previous section. 
Let us first introduce Osterwalder and
Schrader axioms~{\cite{osterwalder_axioms_1973,osterwalder_axioms_1975}} in the stronger variant of Eckmann and Epstein~\cite{eckmann_time_ordered_1979} for the family of distributions $(S_n \in \mathcal{S}'
(\mathbb{R}^{3n}))_{n \in\mathbb{N}_{0}}$.

\begin{description}
  \item[OS0] (Distribution property) It holds $S_0 = 1$. There is a Schwartz
  norm $\| \cdot \|_s$ on $\mathcal{S}' (\mathbb{R}^3)$ and $\beta > 0$ such
  that for all $n \in\mathbb{N}$ and $f_1, \ldots, f_n \in \mathcal{S}
  (\mathbb{R}^3)$
  \begin{equation}
    | S_n (f_1 \otimes \ldots \otimes f_n) | \leqslant (n!)^{\beta} \prod_{i =
    1}^n \| f_i \|_s . \label{eq:os-reg}
  \end{equation}
  \item[OS1] (Euclidean invariance) For each $n\in\mathbb{N}$, $g = (a, R) \in
  \mathbb{R}^3 \times \mathrm{O} (3)$, $f_1, \ldots, f_n \in \mathcal{S}
  (\mathbb{R}^3)$
  \[ S_n ((a, R) .f_1 \otimes \ldots \otimes (a, R) .f_n) = S_n (f_1 \otimes
     \ldots \otimes f_n), \]
  where $(a, R) .f_n (x) = f_n (a + R x)$ and where $\mathrm{O}(3)$ is the orthogonal
  group of $\mathbb{R}^3$.
  
  \item[OS2] (Reflection positivity) Let  $\mathbb{R}^{3 n}_{+} = \{ (x^{(1)}, \ldots, x^{(n)}) \in
(\mathbb{R}^3)^n : x_1^{(j)}>0, j=1,\dots,n \}$ and
\[ \mathcal{S}_{\mathbb{C}} (\mathbb{R}^{3 n}_{+}) \assign \{ f \in \mathcal{S}
   (\mathbb{R}^{3 n};\mathbb{C}) : \tmop{supp} (f) \subset \mathbb{R}^{3 n}_{+} \} . \]
 For all sequences $(f_n \in
  \mathcal{S}_{\mathbb{C}} (\mathbb{R}^{3 n}_{+}))_{n \in\mathbb{N}_{0}}$ with
  finitely many nonzero elements
  \begin{equation}
    \sum_{n, m \in\mathbb{N}_{0}} S_{n + m} (\overline{\Theta f_n} \otimes f_m)
    \geqslant 0, \label{eq:OS2}
  \end{equation}
  where $\Theta f_n (x^{(1)}, \ldots, x^{(n)}) = f (\theta x^{(1)}, \ldots,
  \theta x^{(n)})$ and $\theta (x_1, x_{2}, x_3) = (- x_1, x_{2}, x_3)$ is
  the reflection with respect to the plane $x_1 = 0$.
  
  \item[OS3] (Symmetry) For all $n \in\mathbb{N}$, $f_1, \ldots, f_n \in
  \mathcal{S} (\mathbb{R}^3)$ and $\pi$ a permutation of $n$ elements
  \[ S_n (f_1 \otimes \cdots \otimes f_n) = S_n (f_{\pi (1)} \otimes \cdots
     \otimes f_{\pi (n)}) . \]
\end{description}

The reconstruction theorem of Eckmann and Epstein (Theorem 2 and Corollary 3 in~\cite{eckmann_time_ordered_1979}) asserts that
distributions $(S_n)_{n\in\mathbb{N}_{0}}$ which satisfy OS0--3 are the
 Schwinger functions of a uniquely determined system of time-ordered products of relativistic quantum fields.  Note that if Euclidean invariance in OS1 is replaced with translation invariance with respect to the first coordinate (the Euclidean time), then the reconstruction theorem gives anyway a quantum theory with a unitary time evolution, possibly lacking the full Poincar\'e invariance.
  
For any measure $\mu$ on $\mathcal{S}' (\mathbb{R}^3)$ we define $S_n^{\mu}
\in (\mathcal{S}' (\mathbb{R}^3))^{\otimes n}$ as
\[ S_n^{\mu} (f_1 \otimes \cdots \otimes f_n) \assign \int_{\mathcal{S}'
   (\mathbb{R}^3)} \varphi (f_1) \cdots \varphi (f_n) \mu (\mathd \varphi),
   \qquad n\in \mathbb{N}, f_1, \ldots, f_n \in \mathcal{S} (\mathbb{R}^3) . \]
In this case OS3 is trivially satisfied. 
Along this section we will prove
that, for any accumulation point $\nu$, the functions $(S^{\nu}_n)_n$ satisfy
additionally OS0, OS2 and OS1 with the exception of invariance with respect to $\tmop{SO}
(3)$ (but including reflections) and moreover that $\nu$  is not a Gaussian measure.

\subsection{Distribution property}
\label{ss:OS0}

Here we are concerned with proving the bound~{\eqref{eq:os-reg}} for
correlation functions of $\nu$.

\begin{proposition}
  \label{prop:os-reg-bound}There exists $\beta > 1$ and $K > 0$ such that  any limit measure $\nu$ constructed via the procedure in Section~\ref{sec:tight} satisfies: for
  all $n \in\mathbb{N}$ and all $f_1, \ldots, f_n \in H^{1 / 2 +2 \kappa}
  (\rho^{- 2})$ we have
  \[ | \mathbb{E}_{\nu} [\varphi (f_1) \cdots \varphi (f_n)] | \leqslant K^n
     (n!)^{\beta} \prod_{i = 1}^n \| f_i \|_{H^{1 / 2 +2 \kappa} (\rho^{- 2})}
     . \]
     In particular, it satisfies {\em OS0}.
\end{proposition}

\begin{proof}
  For any $\alpha \in (0, 1)$ and any $n \in\mathbb{N}$ we obtain with  the notation $\langle
\varphi \rangle_{\ast} \assign (1 + \| 
\varphi \|_{H^{- 1 / 2 - 2 \kappa}(\rho^{2})}^2)^{1 / 2}$
  \[ \mathbb{E}_{\nu} [\| \varphi \|_{H^{- 1 / 2 -2 \kappa} (\rho^2)}^n] \leqslant
     \mathbb{E}_{\nu} [\langle \varphi \rangle^{\alpha (n / \alpha)}] \leqslant
     \mathbb{E}_{\nu} [\langle \varphi \rangle^{\alpha \lceil n / \alpha \rceil}]
     \leqslant \beta^{- \lceil n / \alpha \rceil} (\lceil n / \alpha \rceil !)
     \mathbb{E}_{\nu} [e^{\beta \langle \varphi \rangle^{\alpha}}] \]
  \[ \leqslant K^n (n!)^{1 / \alpha} \mathbb{E}_{\nu} [e^{\beta \langle \varphi
     \rangle^{\alpha}}] ,\]
  where we used the fact that Stirling's asymptotic approximation of the
  factorial allows to estimate
  \[ \lceil n / \alpha \rceil ! \leqslant C \left( \frac{\lceil n / \alpha
     \rceil}{e} \right)^{\lceil n / \alpha \rceil} (2 \pi \lceil n / \alpha
     \rceil)^{1 / 2} \leqslant C \left( \frac{2 (n / \alpha)}{e} \right)^{n /
     \alpha + 1} (2 \pi \lceil n / \alpha \rceil)^{1 / 2} \]
  \[ \leqslant K^n \left[ \left( \frac{n}{e} \right)^n (2 \pi n)^{1 / 2}
     \right]^{1 / \alpha} \leqslant K^n (n!)^{1 / \alpha} \]
  for some constants $C, K$, uniformly in $n$ (we allow $K$ to change from
  line to line). From this we can conclude using Proposition \ref{lemma:int-bound}.
\end{proof}

\subsection{Translation invariance}
\label{ss:OS1}

For $h \in \mathbb{R}^3$ we denote by $\mathcal{T}_h : \mathcal{S}'
(\mathbb{R}^3) \rightarrow \mathcal{S}' (\mathbb{R}^3)$ the translation
operator, namely, $\mathcal{T}_h f (x) \assign f (x - h)$. Analogically, for a
measure $\mu$ on $\mathcal{S}' (\mathbb{R}^3)$ we define its translation by
$\mathcal{T}_h \mu (F) \assign \mu (F \circ \mathcal{T}_h)$ where $F \in C_b
(\mathcal{S}' (\mathbb{R}^3))$. We say that $\mu$ is translation invariant if
for all $h \in \mathbb{R}^3$ it holds $\mathcal{T}_h \mu = \mu$.

\begin{proposition}
 Any limit measure $\nu$ constructed via the procedure in Section~\ref{sec:tight} is translation
  invariant.
\end{proposition}

\begin{proof}
  By their definition in {\eqref{eq:gibbs}}, the approximate
  measures $\nu_{M, \varepsilon}$ are translation invariant under lattice
  shifts. That is, for $h_{\varepsilon} \in \Lambda_{\varepsilon}$ it holds
  $\mathcal{T}_{h_{\varepsilon}} \nu_{M, \varepsilon} = \nu_{M, \varepsilon}$.
  In other words, the processes $\varphi_{M, \varepsilon}$ and
  $\mathcal{T}_{h_{\varepsilon}} \varphi_{M, \varepsilon}$ coincide in law. In
  addition, since the translation $\mathcal{T}_{h_{\varepsilon}}$ commutes
  with the extension operator $\mathcal{E}^{\varepsilon}$, it follows that
  $\mathcal{E}^{\varepsilon} \varphi_{M, \varepsilon}$ and
  $\mathcal{T}_{h_{\varepsilon}} \mathcal{E}^{\varepsilon} \varphi_{M,
  \varepsilon}$ coincide in law. Now we recall that the limiting measure $\nu$
  was obtained as a weak limit of the laws of $\mathcal{E}^{\varepsilon}
  \varphi_{M, \varepsilon}$ on $H^{- 1 / 2 - 2\kappa} (\rho^{2 + \gamma})$. If
  $h \in \mathbb{R}^d$ is given, there exists a sequence $h_{\varepsilon} \in
  \Lambda_{\varepsilon}$ such that $h_{\varepsilon} \rightarrow h$. Let
  $\kappa \in (0, 1)$ be small and arbitrary. Then we have for $F \in C^{0,
  1}_b (H^{- 1 / 2 - 3 \kappa} (\rho^{2 + \gamma}))$ that
  \[ \mathcal{T}_h \nu (F) = \nu (F \circ \mathcal{T}_h) = \lim_{\varepsilon
     \rightarrow 0, M \rightarrow \infty} \mathbb{P} \circ
     (\mathcal{E}^{\varepsilon} \varphi_{M, \varepsilon})^{- 1} (F \circ
     \mathcal{T}_h) = \lim_{\varepsilon \rightarrow 0, M \rightarrow \infty}
     \mathbb{E} [F (\mathcal{T}_h \mathcal{E}^{\varepsilon} \varphi_{M,
     \varepsilon})] \]
  \[ \  \]
  \[ = \lim_{\varepsilon \rightarrow 0, M \rightarrow \infty} \mathbb{E} [F
     (\mathcal{T}_{h_{\varepsilon}} \mathcal{E}^{\varepsilon} \varphi_{M,
     \varepsilon})] = \lim_{\varepsilon \rightarrow 0, M \rightarrow \infty}
     \mathbb{E} [F (\mathcal{E}^{\varepsilon} \varphi_{M, \varepsilon})] = \nu
     (F), \]
  where in the third inequality we used the regularity of $F$ and Theorem~\ref{thm:tight} as follows
  \[ \mathbb{E} [F (\mathcal{T}_h \mathcal{E}^{\varepsilon} \varphi_{M,
     \varepsilon}) - F (\mathcal{T}_{h_{\varepsilon}}
     \mathcal{E}^{\varepsilon} \varphi_{M, \varepsilon})] \leqslant \| F
     \|_{C^{0, 1}_b} \mathbb{E} \| \mathcal{T}_h \mathcal{E}^{\varepsilon}
     \varphi_{M, \varepsilon} - \mathcal{T}_{h_{\varepsilon}}
     \mathcal{E}^{\varepsilon} \varphi_{M, \varepsilon} \|_{H^{- 1 / 2 - 3
     \kappa} (\rho^{2 + \gamma})} \]
  \[ \lesssim (h - h_{\varepsilon})^{\kappa} \mathbb{E} \|
     \mathcal{E}^{\varepsilon} \varphi_{M, \varepsilon} \|_{H^{- 1 / 2 -
     2\kappa} (\rho^{2 + \gamma})} \lesssim (h - h_{\varepsilon})^{\kappa}
     \rightarrow 0 \quad \tmop{as} \quad \varepsilon \rightarrow 0. \]
  If $F \in C_b (H^{- 1 / 2 - 3 \kappa} (\rho^{2 + \gamma}))$, then by
  approximation and dominated convergence theorem we also get $\mathcal{T}_h
  \nu (F) = \nu (F)$, which completes the proof.
\end{proof}

\subsection{Reflection positivity}
\label{ss:OS2}

As the next step we
establish reflection positivity of
$\nu$ with respect to the reflection given by any of the hyperplanes $\{x_i=0\}\subset \mathbb{R}^3$ for $i \in \{1, 2,3\}$. 
Fix a small $\delta>0$  and $i\in\{1,2,3\}$ and  define
the space of functionals depending on fields restricted to
$\mathbb{R}^3_{+, \delta}  := \{x \in \mathbb{R}^3 ; x_i > \delta\}$, $\delta\geqslant0$, by
\[ \mathcal{H}_{+, \delta} \assign \left\{ \sum_{k = 1}^K c_k e^{i \varphi
   (f_k)} ; c_k \in \mathbb{C}, f_k \in C^{\infty}_0 (\mathbb{R}^3_{+,\delta}), K \in
   \mathbb{N} \right\} \]
and let $\mathcal{H}_+ =\mathcal{H}_{+, 0}$. For a function $f : \mathbb{R}^3
\rightarrow \mathbb{R}$ we define its reflection
\[ (\theta f) (x) \assign (\theta^i f) (x) \assign f (x_1, \ldots, x_{i - 1},
   - x_i, x_{i + 1}, \ldots, x_3) \]
and extend it to $F \in \mathcal{H}_+$ by $\theta F (\varphi (f_1), \ldots,
\varphi (f_K)) \assign F (\varphi (\theta f_1), \ldots, \varphi (\theta
f_K))$. Hence for $F \in \mathcal{H}_{+, \delta}$ the reflection $\theta F$
depends on $\varphi$ evaluated at $x \in \mathbb{R}^3$ with $x_i < - \delta$.

A measure $\mu$ on $\mathcal{S}'(\mathbb{R}^3)$ is \emph{reflection positive} if
\[ \mathbb{E}_{\mu}  [\overline{\theta F} F] = \int_{\mathcal{S}'
   (\mathbb{R}^3)} \overline{\theta F (\varphi)} F (\varphi) \mu (\mathd
   \varphi) 
   \geqslant 0, \]
for all $F = \sum_{k = 1}^K c_k e^{i \varphi (f_k)} \in \mathcal{H}_+$.
A similar definition applies to measures on functions on the periodic lattice ${\Lambda_{M, \varepsilon}}$ replacing the space $\mathcal{H}_+$ with the appropriate modification $\mathcal{H}_+^{M, \varepsilon}$ given by
$$
\mathcal{H}^{M,\varepsilon}_{+} \assign \left\{ \sum_{k = 1}^K c_k e^{i \varphi
   (f_k)} ; c_k \in \mathbb{C}, f_k :\Lambda_{M,\varepsilon}\cap\mathbb{R}^{3}_{+,0}\to\mathbb{R}\right\}.
$$
The reflection $\theta$ is then defined as on the full space.
Here and also in the proof of Proposition~\ref{prop:RP} below we implicitly assume that $\varepsilon$ is small enough and $M$ is large enough.

An important fact is that for every  $\varepsilon, M$ the Gibbs  measures $\nu_{M, \varepsilon}$ are reflection
  positive see~{\cite[Theorem 7.10.3]{MR887102} or  \cite[Lemma 10.8]{friedli2017statistical}}. The key point of the next proposition is that this property is preserved along the passage to the limit $M\to\infty$, $\varepsilon\to0$.

\begin{proposition}\label{prop:RP}
  Any limit measure $\nu$ constructed via the procedure in
  Section~\ref{sec:tight} is reflection positive with respect to all
  reflections $\theta = \theta^i$, $i \in \{1, 2,3\}$. In particular, its correlation functions
  satisfy {\em OS2}.
\end{proposition}

\begin{proof}
  We recall that the measure $\nu$ was obtained as a
  limit of suitable continuum extensions of the measures $\nu_{M,
  \varepsilon}$ given by {\eqref{eq:gibbs}}. 
 Therefore, up to a subsequence, we have
  \[ \mathbb{E}_{\nu}  [\overline{\theta F} F]  = \lim_{\varepsilon
     \rightarrow 0, M \rightarrow \infty} \mathbb{E} [\overline{F (\theta
     \mathcal{E}^{\varepsilon} \varphi_{M, \varepsilon})} F
     (\mathcal{E}^{\varepsilon} \varphi_{M, \varepsilon})] . \]
 Recall  that the function $w$ in the definition of the
  extension operator $\mathcal{E}^{\varepsilon}$ is radially
  symmetric. Hence, we have
$
(\theta
     \mathcal{E}^{\varepsilon} \varphi_{M, \varepsilon})(f)= 
     \varphi_{M, \varepsilon}(\mathcal{E}^{\varepsilon,*} \theta f) =\varphi_{M, \varepsilon}(\theta \mathcal{E}^{\varepsilon,*}  f)
$  
for any function $f\in C^\infty_0(\mathbb{R}^3)$  supported in $\{x\in\mathbb{R}^{3};|x_i| <  M/2 -\delta\}$. Here $\mathcal{E}^{\varepsilon,*}$ is the adjoint of the extension operator. For a fixed $F  \in \mathcal{H}_{+, \delta}$ we have therefore $F (\theta
     \mathcal{E}^{\varepsilon} \varphi_{M, \varepsilon})=(F \circ \mathcal{E}^{\varepsilon})(
      \theta\varphi_{M, \varepsilon})$  provided $\varepsilon$ is small enough and $M$ large enough depending  on $F$ and
  $\delta$. Hence,
    \[ \mathbb{E}_{\nu}  [\overline{\theta F} F] = \lim_{\varepsilon \rightarrow
     0, M \rightarrow \infty} \mathbb{E} [\overline{F
     (\mathcal{E}^{\varepsilon} \theta \varphi_{M, \varepsilon})} F
     (\mathcal{E}^{\varepsilon} \varphi_{M, \varepsilon})]. \]
     However, since the extension operator is defined as a convolution with a non-compactly supported function $w^{\varepsilon}$, it is generally not true that $F \circ \mathcal{E}^{\varepsilon}
  \in \mathcal{H}_{+}^{M, \varepsilon}$. Thus, in order to be able to use the reflection positivity of the measures $\nu_{M,\varepsilon}$, we need to introduce an additional cut-off: let  $H_{\delta}:\mathbb{R}^{3}\to[0,1]$ be smooth and supported on $\mathbb{R}^{3}_{+,0}$ such that $H_{\delta}=1$ on $\mathbb{R}^{3}_{+,\delta/2}$. We denote by $H_{\delta,\varepsilon}$ its restriction to $\Lambda_{\varepsilon}$ and write
  \[ 
 R_\varepsilon  :=  F(\mathcal{E}^{\varepsilon}\varphi_{M,\varepsilon})-  F(\mathcal{E}^{\varepsilon}(H_{
  \delta,\varepsilon}\varphi_{M,\varepsilon}))
  .\] 
  Our goal is to show that $R_\varepsilon$ vanishes a.s. as $\varepsilon\to 0$. In view of the fact that $F$ is cylindrical and then regularity of $\varphi_{M,\varepsilon}$, it is enough to show that
  \begin{equation}
  \label{eq:H-lim}
\lim_{\varepsilon\to 0} \| (1-H_{\delta,\varepsilon})\mathcal{E}^{\varepsilon,*}f\|_{H^{1/2+\kappa,\varepsilon}(\rho^{-2})}=0
\end{equation}
for any function $f\in C^{\infty}_{0}(\mathbb{R}^{3}_{+,\delta}).$ It holds
  \begin{equation}\label{eq:lll}
  [(1-H_{\delta,\varepsilon})\mathcal{E}^{\varepsilon,*}f](x)=(1-H_{\delta,\varepsilon})(x)\int_{y\in \mathbb{R}^{3}:y_{i}>\delta}w^{\varepsilon}(x-y)f(y)\mathrm{d} y,
  \end{equation}
  where $1-H_{\delta,\varepsilon}(x)\neq 0$ only when $x_{i}\leq \delta/2$. Since $w^{\varepsilon}(\cdot)=\varepsilon^{-d}w(\varepsilon^{-1}\cdot)$ with $w\in\mathcal{S}(\mathbb{R}^{3}),$ we have for an arbitrary $K>0$ and $m\in\mathbb{N}$
  $$| \nabla^{m}w^{\varepsilon} (x - y) | \lesssim \varepsilon^{- d -m} | \varepsilon^{-1}(x - y)  |^{- K}.$$
In addition,  we know that the relevant $|x-y|$ on the right hand side of \eqref{eq:lll} satisfy $|x_{i}-y_{i}|>\delta/2$. Hence, choosing $K,L$ sufficiently large will give us a decay as $\varepsilon\to 0$ for every fixed $\delta>0$. Indeed, we also have
$
|\nabla_{\varepsilon}^{m}(1-H_{\delta,\varepsilon})(x)|\lesssim \delta^{-1}
$
uniformly in $\varepsilon$.
Thus,   we may estimate 
$$
   \| (1-H_{\delta,\varepsilon})\mathcal{E}^{\varepsilon,*}f \|_{H^{1/2+\kappa}(\rho^{-2})}\leq c(\varepsilon,\delta)\|f\|_{L^{\infty}}
 , $$
where $c(\varepsilon,\delta)\to 0$ as $\varepsilon\to\infty$ for every fixed  $\delta>0$. This concludes the proof of~\eqref{eq:H-lim}.

 On the other hand,   $F (\mathcal{E}^{\varepsilon} ( H_{\delta,\varepsilon} \cdot))
  \in \mathcal{H}_{+}^{M, \varepsilon}$  and consequently
  \[ \mathbb{E}_{\nu}  [\overline{\theta F} F] = \lim_{\varepsilon \rightarrow
     0, M \rightarrow \infty} \mathbb{E} [\overline{F
     (\mathcal{E}^{\varepsilon} (H_{\delta,\varepsilon}\theta \varphi_{M, \varepsilon}))} F
     (\mathcal{E}^{\varepsilon}(H_{\delta,\varepsilon} \varphi_{M, \varepsilon}))] \]
  \[ = \lim_{\varepsilon \rightarrow 0, M \rightarrow \infty} \mathbb{E}_{\nu_{M, \varepsilon}}
     [\overline{\theta (F (\mathcal{E}^{\varepsilon}( H_{\delta,\varepsilon} \cdot)))} F (\mathcal{E}^{\varepsilon}( H_{\delta,\varepsilon} \cdot))] \geqslant 0, \]
  where we used the reflection positivity  of the  measure $\nu_{M, \varepsilon}$.
  Using the support properties of $\nu$ we can now approximate any $F \in
  \mathcal{H}_+$ by functions in $\mathcal{H}_{+, \delta}$ and therefore
  obtain the first claim. Let us now show that~{\eqref{eq:OS2}} holds.
  Thanks to the exponential integrability satisfied by $\nu$, any
  polynomial of the form $G = \sum_{n \in\mathbb{N}_{0}} \varphi^{\otimes n} (f_n)$
  for sequences $(f_n \in \mathcal{S}_{\mathbb{C}} (\mathbb{R}^{3 n}_{+}))_{n
  \in\mathbb{N}_{0}}$ with finitely many nonzero elements, belongs to $L^2 (\nu)$.
  In particular it can be approximated in $L^2 (\nu)$ by  a sequence
  $(F_n)_n$ of cylinder functions in $\mathcal{H}_+$. Therefore
  $\mathbb{E}_{\nu}  [\overline{\theta G} G] = \lim_{n \rightarrow \infty}
  \mathbb{E}_{\nu}  [\overline{\theta F_n} F_n] \geqslant 0$ and we conclude
  that
  \[ \sum_{n, m \in\mathbb{N}_{0}} S_{n + m}^{\nu} (\overline{\theta f_n} \otimes
     f_m) = \sum_{n, m \in\mathbb{N}_{0}} \mathbb{E}_{\nu} [\varphi^{\otimes n}
     (\overline{\theta f_n}) \varphi^{\otimes m} (f_m)] =\mathbb{E}_{\nu} 
     [\overline{\theta G} G] \geqslant 0. \]
\end{proof}

\subsection{Non-Gaussianity}
\label{ss:nonG}

\begin{theorem}
 If $\lambda> 0$ then any limit measure $\nu$ constructed via the procedure in Section~\ref{sec:tight} is non-Gaussian.
\end{theorem}

\begin{proof}
  In order to show that the limiting measure $\nu$ is non-Gaussian, it is
  sufficient to prove that the connected four-point function is nonzero, see \cite{MR723546}. In other words, we shall prove that the
  distribution
  \[ U^{\nu}_4 (x_1, \ldots, x_4) \assign \mathbb{E}_{\nu} [\varphi (x_1)
     \cdots \varphi (x_4)] \]
  \[ -\mathbb{E}_{\nu} [\varphi (x_1) \varphi (x_2)] \mathbb{E}_{\nu} [\varphi
     (x_3) \varphi (x_4)] -\mathbb{E}_{\nu} [\varphi (x_1) \varphi (x_3)]
     \mathbb{E}_{\nu} [\varphi (x_2) \varphi (x_4)] \]
  \[ -\mathbb{E}_{\nu} [\varphi (x_1) \varphi (x_4)] \mathbb{E}_{\nu} [\varphi
     (x_2) \varphi (x_3)], \qquad x_1, \ldots, x_4 \in \mathbb{R}^d, \]
  is nonzero.
  
  Recall that in Theorem~\ref{thm:main} we obtained a limit measure $\mu$ which is the joint law of $(\varphi,X,X^{\!\resizebox{0.6em}{!}{
\begin{tikzpicture}
\pgfpathmoveto{\pgfqpoint{0cm}{-0.035cm}}
\pgfpathlineto{\pgfqpoint{1.376cm}{-0.035cm}}
\pgfpathlineto{\pgfqpoint{1.376cm}{1.552cm}}
\pgfpathlineto{\pgfqpoint{0cm}{1.552cm}}
\pgfpathclose
\pgfusepath{clip}
\begin{pgfscope}
\begin{pgfscope}
\pgfpathmoveto{\pgfqpoint{0cm}{-0.035cm}}
\pgfpathlineto{\pgfqpoint{1.376cm}{-0.035cm}}
\pgfpathlineto{\pgfqpoint{1.376cm}{1.552cm}}
\pgfpathlineto{\pgfqpoint{0cm}{1.552cm}}
\pgfpathclose
\pgfusepath{clip}
\begin{pgfscope}
\begin{pgfscope}
\pgfsetdash{}{0cm}
\pgfsetlinewidth{0.818mm}
\pgfsetroundcap
\pgfsetroundjoin
\pgfsetmiterlimit{7.0}
\definecolor{eps2pgf_color}{gray}{0}\pgfsetstrokecolor{eps2pgf_color}\pgfsetfillcolor{eps2pgf_color}
\pgfpathmoveto{\pgfqpoint{0.117cm}{1.421cm}}
\pgfpathlineto{\pgfqpoint{0.682cm}{0.671cm}}
\pgfpathlineto{\pgfqpoint{1.246cm}{1.421cm}}
\pgfusepath{stroke}
\end{pgfscope}
\definecolor{eps2pgf_color}{gray}{0}\pgfsetstrokecolor{eps2pgf_color}\pgfsetfillcolor{eps2pgf_color}
\pgfpathmoveto{\pgfqpoint{0.273cm}{1.395cm}}
\pgfpathcurveto{\pgfqpoint{0.273cm}{1.432cm}}{\pgfqpoint{0.259cm}{1.467cm}}{\pgfqpoint{0.233cm}{1.492cm}}
\pgfpathcurveto{\pgfqpoint{0.207cm}{1.518cm}}{\pgfqpoint{0.173cm}{1.532cm}}{\pgfqpoint{0.137cm}{1.532cm}}
\pgfpathcurveto{\pgfqpoint{0.1cm}{1.532cm}}{\pgfqpoint{0.066cm}{1.518cm}}{\pgfqpoint{0.04cm}{1.492cm}}
\pgfpathcurveto{\pgfqpoint{0.014cm}{1.467cm}}{\pgfqpoint{0cm}{1.432cm}}{\pgfqpoint{0cm}{1.395cm}}
\pgfpathcurveto{\pgfqpoint{0cm}{1.359cm}}{\pgfqpoint{0.014cm}{1.324cm}}{\pgfqpoint{0.04cm}{1.299cm}}
\pgfpathcurveto{\pgfqpoint{0.066cm}{1.273cm}}{\pgfqpoint{0.1cm}{1.258cm}}{\pgfqpoint{0.137cm}{1.258cm}}
\pgfpathcurveto{\pgfqpoint{0.173cm}{1.258cm}}{\pgfqpoint{0.207cm}{1.273cm}}{\pgfqpoint{0.233cm}{1.299cm}}
\pgfpathcurveto{\pgfqpoint{0.259cm}{1.324cm}}{\pgfqpoint{0.273cm}{1.359cm}}{\pgfqpoint{0.273cm}{1.395cm}}
\pgfusepath{fill}
\begin{pgfscope}
\pgfsetdash{}{0cm}
\pgfsetlinewidth{0.818mm}
\pgfsetmiterlimit{7.0}
\pgfpathmoveto{\pgfqpoint{0.682cm}{0.671cm}}
\pgfpathlineto{\pgfqpoint{0.679cm}{1.418cm}}
\pgfusepath{stroke}
\end{pgfscope}
\pgfpathmoveto{\pgfqpoint{0.815cm}{1.399cm}}
\pgfpathcurveto{\pgfqpoint{0.815cm}{1.435cm}}{\pgfqpoint{0.801cm}{1.47cm}}{\pgfqpoint{0.775cm}{1.496cm}}
\pgfpathcurveto{\pgfqpoint{0.75cm}{1.521cm}}{\pgfqpoint{0.715cm}{1.536cm}}{\pgfqpoint{0.679cm}{1.536cm}}
\pgfpathcurveto{\pgfqpoint{0.643cm}{1.536cm}}{\pgfqpoint{0.608cm}{1.521cm}}{\pgfqpoint{0.582cm}{1.496cm}}
\pgfpathcurveto{\pgfqpoint{0.557cm}{1.47cm}}{\pgfqpoint{0.542cm}{1.435cm}}{\pgfqpoint{0.542cm}{1.399cm}}
\pgfpathcurveto{\pgfqpoint{0.542cm}{1.363cm}}{\pgfqpoint{0.557cm}{1.328cm}}{\pgfqpoint{0.582cm}{1.302cm}}
\pgfpathcurveto{\pgfqpoint{0.608cm}{1.276cm}}{\pgfqpoint{0.643cm}{1.262cm}}{\pgfqpoint{0.679cm}{1.262cm}}
\pgfpathcurveto{\pgfqpoint{0.715cm}{1.262cm}}{\pgfqpoint{0.75cm}{1.276cm}}{\pgfqpoint{0.775cm}{1.302cm}}
\pgfpathcurveto{\pgfqpoint{0.801cm}{1.328cm}}{\pgfqpoint{0.815cm}{1.363cm}}{\pgfqpoint{0.815cm}{1.399cm}}
\pgfusepath{fill}
\pgfpathmoveto{\pgfqpoint{1.345cm}{1.371cm}}
\pgfpathcurveto{\pgfqpoint{1.345cm}{1.408cm}}{\pgfqpoint{1.331cm}{1.442cm}}{\pgfqpoint{1.305cm}{1.468cm}}
\pgfpathcurveto{\pgfqpoint{1.28cm}{1.494cm}}{\pgfqpoint{1.245cm}{1.508cm}}{\pgfqpoint{1.209cm}{1.508cm}}
\pgfpathcurveto{\pgfqpoint{1.172cm}{1.508cm}}{\pgfqpoint{1.138cm}{1.494cm}}{\pgfqpoint{1.112cm}{1.468cm}}
\pgfpathcurveto{\pgfqpoint{1.087cm}{1.442cm}}{\pgfqpoint{1.072cm}{1.408cm}}{\pgfqpoint{1.072cm}{1.371cm}}
\pgfpathcurveto{\pgfqpoint{1.072cm}{1.335cm}}{\pgfqpoint{1.087cm}{1.3cm}}{\pgfqpoint{1.112cm}{1.274cm}}
\pgfpathcurveto{\pgfqpoint{1.138cm}{1.249cm}}{\pgfqpoint{1.172cm}{1.234cm}}{\pgfqpoint{1.209cm}{1.234cm}}
\pgfpathcurveto{\pgfqpoint{1.245cm}{1.234cm}}{\pgfqpoint{1.28cm}{1.249cm}}{\pgfqpoint{1.305cm}{1.274cm}}
\pgfpathcurveto{\pgfqpoint{1.331cm}{1.3cm}}{\pgfqpoint{1.345cm}{1.335cm}}{\pgfqpoint{1.345cm}{1.371cm}}
\pgfusepath{fill}
\begin{pgfscope}
\pgfsetdash{}{0cm}
\pgfsetlinewidth{0.818mm}
\pgfsetroundcap
\pgfsetmiterlimit{4.0}
\pgfpathmoveto{\pgfqpoint{0.682cm}{0.671cm}}
\pgfpathlineto{\pgfqpoint{0.682cm}{0.042cm}}
\pgfusepath{stroke}
\end{pgfscope}
\end{pgfscope}
\end{pgfscope}
\end{pgfscope}
\end{tikzpicture}}})$ and that $\nu$ is the marginal corresponding to the first component. Let $K_i = \mathcal{F}^{- 1} \varphi_i$ be a Littlewood--Paley
  projector and consider the connected four-point function $U^{\nu}_4$
  convolved with $(K_i, K_i, K_i, K_i)$ and evaluated at $(x_{1},\dots,x_{4})=(0,\dots, 0)$, that
  is,
  \[ U^{\nu}_4 \ast (K_i, K_i, K_i, K_i) (0, 0, 0, 0) =\mathbb{E}_{\nu}
     [(\Delta_i \varphi)^4 (0)] - 3\mathbb{E}_{\nu} [(\Delta_i \varphi)^2
     (0)]^2 \]
  \[ =\mathbb{E}_{\mu} [(\Delta_i \varphi)^4 (0)] - 3\mathbb{E}_{\mu}
     [(\Delta_i \varphi)^2 (0)]^2 \backassign L (\varphi, \varphi, \varphi,
     \varphi), \]
  where $L$ is a quadrilinear form. Since under the limit $\mu$ we have the
   decomposition $\varphi = X -\lambda X^{\!\resizebox{0.6em}{!}{
\begin{tikzpicture}
\pgfpathmoveto{\pgfqpoint{0cm}{-0.035cm}}
\pgfpathlineto{\pgfqpoint{1.376cm}{-0.035cm}}
\pgfpathlineto{\pgfqpoint{1.376cm}{1.552cm}}
\pgfpathlineto{\pgfqpoint{0cm}{1.552cm}}
\pgfpathclose
\pgfusepath{clip}
\begin{pgfscope}
\begin{pgfscope}
\pgfpathmoveto{\pgfqpoint{0cm}{-0.035cm}}
\pgfpathlineto{\pgfqpoint{1.376cm}{-0.035cm}}
\pgfpathlineto{\pgfqpoint{1.376cm}{1.552cm}}
\pgfpathlineto{\pgfqpoint{0cm}{1.552cm}}
\pgfpathclose
\pgfusepath{clip}
\begin{pgfscope}
\begin{pgfscope}
\pgfsetdash{}{0cm}
\pgfsetlinewidth{0.818mm}
\pgfsetroundcap
\pgfsetroundjoin
\pgfsetmiterlimit{7.0}
\definecolor{eps2pgf_color}{gray}{0}\pgfsetstrokecolor{eps2pgf_color}\pgfsetfillcolor{eps2pgf_color}
\pgfpathmoveto{\pgfqpoint{0.117cm}{1.421cm}}
\pgfpathlineto{\pgfqpoint{0.682cm}{0.671cm}}
\pgfpathlineto{\pgfqpoint{1.246cm}{1.421cm}}
\pgfusepath{stroke}
\end{pgfscope}
\definecolor{eps2pgf_color}{gray}{0}\pgfsetstrokecolor{eps2pgf_color}\pgfsetfillcolor{eps2pgf_color}
\pgfpathmoveto{\pgfqpoint{0.273cm}{1.395cm}}
\pgfpathcurveto{\pgfqpoint{0.273cm}{1.432cm}}{\pgfqpoint{0.259cm}{1.467cm}}{\pgfqpoint{0.233cm}{1.492cm}}
\pgfpathcurveto{\pgfqpoint{0.207cm}{1.518cm}}{\pgfqpoint{0.173cm}{1.532cm}}{\pgfqpoint{0.137cm}{1.532cm}}
\pgfpathcurveto{\pgfqpoint{0.1cm}{1.532cm}}{\pgfqpoint{0.066cm}{1.518cm}}{\pgfqpoint{0.04cm}{1.492cm}}
\pgfpathcurveto{\pgfqpoint{0.014cm}{1.467cm}}{\pgfqpoint{0cm}{1.432cm}}{\pgfqpoint{0cm}{1.395cm}}
\pgfpathcurveto{\pgfqpoint{0cm}{1.359cm}}{\pgfqpoint{0.014cm}{1.324cm}}{\pgfqpoint{0.04cm}{1.299cm}}
\pgfpathcurveto{\pgfqpoint{0.066cm}{1.273cm}}{\pgfqpoint{0.1cm}{1.258cm}}{\pgfqpoint{0.137cm}{1.258cm}}
\pgfpathcurveto{\pgfqpoint{0.173cm}{1.258cm}}{\pgfqpoint{0.207cm}{1.273cm}}{\pgfqpoint{0.233cm}{1.299cm}}
\pgfpathcurveto{\pgfqpoint{0.259cm}{1.324cm}}{\pgfqpoint{0.273cm}{1.359cm}}{\pgfqpoint{0.273cm}{1.395cm}}
\pgfusepath{fill}
\begin{pgfscope}
\pgfsetdash{}{0cm}
\pgfsetlinewidth{0.818mm}
\pgfsetmiterlimit{7.0}
\pgfpathmoveto{\pgfqpoint{0.682cm}{0.671cm}}
\pgfpathlineto{\pgfqpoint{0.679cm}{1.418cm}}
\pgfusepath{stroke}
\end{pgfscope}
\pgfpathmoveto{\pgfqpoint{0.815cm}{1.399cm}}
\pgfpathcurveto{\pgfqpoint{0.815cm}{1.435cm}}{\pgfqpoint{0.801cm}{1.47cm}}{\pgfqpoint{0.775cm}{1.496cm}}
\pgfpathcurveto{\pgfqpoint{0.75cm}{1.521cm}}{\pgfqpoint{0.715cm}{1.536cm}}{\pgfqpoint{0.679cm}{1.536cm}}
\pgfpathcurveto{\pgfqpoint{0.643cm}{1.536cm}}{\pgfqpoint{0.608cm}{1.521cm}}{\pgfqpoint{0.582cm}{1.496cm}}
\pgfpathcurveto{\pgfqpoint{0.557cm}{1.47cm}}{\pgfqpoint{0.542cm}{1.435cm}}{\pgfqpoint{0.542cm}{1.399cm}}
\pgfpathcurveto{\pgfqpoint{0.542cm}{1.363cm}}{\pgfqpoint{0.557cm}{1.328cm}}{\pgfqpoint{0.582cm}{1.302cm}}
\pgfpathcurveto{\pgfqpoint{0.608cm}{1.276cm}}{\pgfqpoint{0.643cm}{1.262cm}}{\pgfqpoint{0.679cm}{1.262cm}}
\pgfpathcurveto{\pgfqpoint{0.715cm}{1.262cm}}{\pgfqpoint{0.75cm}{1.276cm}}{\pgfqpoint{0.775cm}{1.302cm}}
\pgfpathcurveto{\pgfqpoint{0.801cm}{1.328cm}}{\pgfqpoint{0.815cm}{1.363cm}}{\pgfqpoint{0.815cm}{1.399cm}}
\pgfusepath{fill}
\pgfpathmoveto{\pgfqpoint{1.345cm}{1.371cm}}
\pgfpathcurveto{\pgfqpoint{1.345cm}{1.408cm}}{\pgfqpoint{1.331cm}{1.442cm}}{\pgfqpoint{1.305cm}{1.468cm}}
\pgfpathcurveto{\pgfqpoint{1.28cm}{1.494cm}}{\pgfqpoint{1.245cm}{1.508cm}}{\pgfqpoint{1.209cm}{1.508cm}}
\pgfpathcurveto{\pgfqpoint{1.172cm}{1.508cm}}{\pgfqpoint{1.138cm}{1.494cm}}{\pgfqpoint{1.112cm}{1.468cm}}
\pgfpathcurveto{\pgfqpoint{1.087cm}{1.442cm}}{\pgfqpoint{1.072cm}{1.408cm}}{\pgfqpoint{1.072cm}{1.371cm}}
\pgfpathcurveto{\pgfqpoint{1.072cm}{1.335cm}}{\pgfqpoint{1.087cm}{1.3cm}}{\pgfqpoint{1.112cm}{1.274cm}}
\pgfpathcurveto{\pgfqpoint{1.138cm}{1.249cm}}{\pgfqpoint{1.172cm}{1.234cm}}{\pgfqpoint{1.209cm}{1.234cm}}
\pgfpathcurveto{\pgfqpoint{1.245cm}{1.234cm}}{\pgfqpoint{1.28cm}{1.249cm}}{\pgfqpoint{1.305cm}{1.274cm}}
\pgfpathcurveto{\pgfqpoint{1.331cm}{1.3cm}}{\pgfqpoint{1.345cm}{1.335cm}}{\pgfqpoint{1.345cm}{1.371cm}}
\pgfusepath{fill}
\begin{pgfscope}
\pgfsetdash{}{0cm}
\pgfsetlinewidth{0.818mm}
\pgfsetroundcap
\pgfsetmiterlimit{4.0}
\pgfpathmoveto{\pgfqpoint{0.682cm}{0.671cm}}
\pgfpathlineto{\pgfqpoint{0.682cm}{0.042cm}}
\pgfusepath{stroke}
\end{pgfscope}
\end{pgfscope}
\end{pgfscope}
\end{pgfscope}
\end{tikzpicture}}} + \zeta$, we may write
  \begin{equation}
    L (\varphi, \varphi, \varphi, \varphi) = L (X, X, X, X) - 4 \lambda L ( X, X,
    X, X^{\!\resizebox{0.6em}{!}{
\begin{tikzpicture}
\pgfpathmoveto{\pgfqpoint{0cm}{-0.035cm}}
\pgfpathlineto{\pgfqpoint{1.376cm}{-0.035cm}}
\pgfpathlineto{\pgfqpoint{1.376cm}{1.552cm}}
\pgfpathlineto{\pgfqpoint{0cm}{1.552cm}}
\pgfpathclose
\pgfusepath{clip}
\begin{pgfscope}
\begin{pgfscope}
\pgfpathmoveto{\pgfqpoint{0cm}{-0.035cm}}
\pgfpathlineto{\pgfqpoint{1.376cm}{-0.035cm}}
\pgfpathlineto{\pgfqpoint{1.376cm}{1.552cm}}
\pgfpathlineto{\pgfqpoint{0cm}{1.552cm}}
\pgfpathclose
\pgfusepath{clip}
\begin{pgfscope}
\begin{pgfscope}
\pgfsetdash{}{0cm}
\pgfsetlinewidth{0.818mm}
\pgfsetroundcap
\pgfsetroundjoin
\pgfsetmiterlimit{7.0}
\definecolor{eps2pgf_color}{gray}{0}\pgfsetstrokecolor{eps2pgf_color}\pgfsetfillcolor{eps2pgf_color}
\pgfpathmoveto{\pgfqpoint{0.117cm}{1.421cm}}
\pgfpathlineto{\pgfqpoint{0.682cm}{0.671cm}}
\pgfpathlineto{\pgfqpoint{1.246cm}{1.421cm}}
\pgfusepath{stroke}
\end{pgfscope}
\definecolor{eps2pgf_color}{gray}{0}\pgfsetstrokecolor{eps2pgf_color}\pgfsetfillcolor{eps2pgf_color}
\pgfpathmoveto{\pgfqpoint{0.273cm}{1.395cm}}
\pgfpathcurveto{\pgfqpoint{0.273cm}{1.432cm}}{\pgfqpoint{0.259cm}{1.467cm}}{\pgfqpoint{0.233cm}{1.492cm}}
\pgfpathcurveto{\pgfqpoint{0.207cm}{1.518cm}}{\pgfqpoint{0.173cm}{1.532cm}}{\pgfqpoint{0.137cm}{1.532cm}}
\pgfpathcurveto{\pgfqpoint{0.1cm}{1.532cm}}{\pgfqpoint{0.066cm}{1.518cm}}{\pgfqpoint{0.04cm}{1.492cm}}
\pgfpathcurveto{\pgfqpoint{0.014cm}{1.467cm}}{\pgfqpoint{0cm}{1.432cm}}{\pgfqpoint{0cm}{1.395cm}}
\pgfpathcurveto{\pgfqpoint{0cm}{1.359cm}}{\pgfqpoint{0.014cm}{1.324cm}}{\pgfqpoint{0.04cm}{1.299cm}}
\pgfpathcurveto{\pgfqpoint{0.066cm}{1.273cm}}{\pgfqpoint{0.1cm}{1.258cm}}{\pgfqpoint{0.137cm}{1.258cm}}
\pgfpathcurveto{\pgfqpoint{0.173cm}{1.258cm}}{\pgfqpoint{0.207cm}{1.273cm}}{\pgfqpoint{0.233cm}{1.299cm}}
\pgfpathcurveto{\pgfqpoint{0.259cm}{1.324cm}}{\pgfqpoint{0.273cm}{1.359cm}}{\pgfqpoint{0.273cm}{1.395cm}}
\pgfusepath{fill}
\begin{pgfscope}
\pgfsetdash{}{0cm}
\pgfsetlinewidth{0.818mm}
\pgfsetmiterlimit{7.0}
\pgfpathmoveto{\pgfqpoint{0.682cm}{0.671cm}}
\pgfpathlineto{\pgfqpoint{0.679cm}{1.418cm}}
\pgfusepath{stroke}
\end{pgfscope}
\pgfpathmoveto{\pgfqpoint{0.815cm}{1.399cm}}
\pgfpathcurveto{\pgfqpoint{0.815cm}{1.435cm}}{\pgfqpoint{0.801cm}{1.47cm}}{\pgfqpoint{0.775cm}{1.496cm}}
\pgfpathcurveto{\pgfqpoint{0.75cm}{1.521cm}}{\pgfqpoint{0.715cm}{1.536cm}}{\pgfqpoint{0.679cm}{1.536cm}}
\pgfpathcurveto{\pgfqpoint{0.643cm}{1.536cm}}{\pgfqpoint{0.608cm}{1.521cm}}{\pgfqpoint{0.582cm}{1.496cm}}
\pgfpathcurveto{\pgfqpoint{0.557cm}{1.47cm}}{\pgfqpoint{0.542cm}{1.435cm}}{\pgfqpoint{0.542cm}{1.399cm}}
\pgfpathcurveto{\pgfqpoint{0.542cm}{1.363cm}}{\pgfqpoint{0.557cm}{1.328cm}}{\pgfqpoint{0.582cm}{1.302cm}}
\pgfpathcurveto{\pgfqpoint{0.608cm}{1.276cm}}{\pgfqpoint{0.643cm}{1.262cm}}{\pgfqpoint{0.679cm}{1.262cm}}
\pgfpathcurveto{\pgfqpoint{0.715cm}{1.262cm}}{\pgfqpoint{0.75cm}{1.276cm}}{\pgfqpoint{0.775cm}{1.302cm}}
\pgfpathcurveto{\pgfqpoint{0.801cm}{1.328cm}}{\pgfqpoint{0.815cm}{1.363cm}}{\pgfqpoint{0.815cm}{1.399cm}}
\pgfusepath{fill}
\pgfpathmoveto{\pgfqpoint{1.345cm}{1.371cm}}
\pgfpathcurveto{\pgfqpoint{1.345cm}{1.408cm}}{\pgfqpoint{1.331cm}{1.442cm}}{\pgfqpoint{1.305cm}{1.468cm}}
\pgfpathcurveto{\pgfqpoint{1.28cm}{1.494cm}}{\pgfqpoint{1.245cm}{1.508cm}}{\pgfqpoint{1.209cm}{1.508cm}}
\pgfpathcurveto{\pgfqpoint{1.172cm}{1.508cm}}{\pgfqpoint{1.138cm}{1.494cm}}{\pgfqpoint{1.112cm}{1.468cm}}
\pgfpathcurveto{\pgfqpoint{1.087cm}{1.442cm}}{\pgfqpoint{1.072cm}{1.408cm}}{\pgfqpoint{1.072cm}{1.371cm}}
\pgfpathcurveto{\pgfqpoint{1.072cm}{1.335cm}}{\pgfqpoint{1.087cm}{1.3cm}}{\pgfqpoint{1.112cm}{1.274cm}}
\pgfpathcurveto{\pgfqpoint{1.138cm}{1.249cm}}{\pgfqpoint{1.172cm}{1.234cm}}{\pgfqpoint{1.209cm}{1.234cm}}
\pgfpathcurveto{\pgfqpoint{1.245cm}{1.234cm}}{\pgfqpoint{1.28cm}{1.249cm}}{\pgfqpoint{1.305cm}{1.274cm}}
\pgfpathcurveto{\pgfqpoint{1.331cm}{1.3cm}}{\pgfqpoint{1.345cm}{1.335cm}}{\pgfqpoint{1.345cm}{1.371cm}}
\pgfusepath{fill}
\begin{pgfscope}
\pgfsetdash{}{0cm}
\pgfsetlinewidth{0.818mm}
\pgfsetroundcap
\pgfsetmiterlimit{4.0}
\pgfpathmoveto{\pgfqpoint{0.682cm}{0.671cm}}
\pgfpathlineto{\pgfqpoint{0.682cm}{0.042cm}}
\pgfusepath{stroke}
\end{pgfscope}
\end{pgfscope}
\end{pgfscope}
\end{pgfscope}
\end{tikzpicture}}} ) + R \label{eq:L24}
  \end{equation}
  where $R$ contains terms which are at least bilinear in $X^{\!\resizebox{0.6em}{!}{
\begin{tikzpicture}
\pgfpathmoveto{\pgfqpoint{0cm}{-0.035cm}}
\pgfpathlineto{\pgfqpoint{1.376cm}{-0.035cm}}
\pgfpathlineto{\pgfqpoint{1.376cm}{1.552cm}}
\pgfpathlineto{\pgfqpoint{0cm}{1.552cm}}
\pgfpathclose
\pgfusepath{clip}
\begin{pgfscope}
\begin{pgfscope}
\pgfpathmoveto{\pgfqpoint{0cm}{-0.035cm}}
\pgfpathlineto{\pgfqpoint{1.376cm}{-0.035cm}}
\pgfpathlineto{\pgfqpoint{1.376cm}{1.552cm}}
\pgfpathlineto{\pgfqpoint{0cm}{1.552cm}}
\pgfpathclose
\pgfusepath{clip}
\begin{pgfscope}
\begin{pgfscope}
\pgfsetdash{}{0cm}
\pgfsetlinewidth{0.818mm}
\pgfsetroundcap
\pgfsetroundjoin
\pgfsetmiterlimit{7.0}
\definecolor{eps2pgf_color}{gray}{0}\pgfsetstrokecolor{eps2pgf_color}\pgfsetfillcolor{eps2pgf_color}
\pgfpathmoveto{\pgfqpoint{0.117cm}{1.421cm}}
\pgfpathlineto{\pgfqpoint{0.682cm}{0.671cm}}
\pgfpathlineto{\pgfqpoint{1.246cm}{1.421cm}}
\pgfusepath{stroke}
\end{pgfscope}
\definecolor{eps2pgf_color}{gray}{0}\pgfsetstrokecolor{eps2pgf_color}\pgfsetfillcolor{eps2pgf_color}
\pgfpathmoveto{\pgfqpoint{0.273cm}{1.395cm}}
\pgfpathcurveto{\pgfqpoint{0.273cm}{1.432cm}}{\pgfqpoint{0.259cm}{1.467cm}}{\pgfqpoint{0.233cm}{1.492cm}}
\pgfpathcurveto{\pgfqpoint{0.207cm}{1.518cm}}{\pgfqpoint{0.173cm}{1.532cm}}{\pgfqpoint{0.137cm}{1.532cm}}
\pgfpathcurveto{\pgfqpoint{0.1cm}{1.532cm}}{\pgfqpoint{0.066cm}{1.518cm}}{\pgfqpoint{0.04cm}{1.492cm}}
\pgfpathcurveto{\pgfqpoint{0.014cm}{1.467cm}}{\pgfqpoint{0cm}{1.432cm}}{\pgfqpoint{0cm}{1.395cm}}
\pgfpathcurveto{\pgfqpoint{0cm}{1.359cm}}{\pgfqpoint{0.014cm}{1.324cm}}{\pgfqpoint{0.04cm}{1.299cm}}
\pgfpathcurveto{\pgfqpoint{0.066cm}{1.273cm}}{\pgfqpoint{0.1cm}{1.258cm}}{\pgfqpoint{0.137cm}{1.258cm}}
\pgfpathcurveto{\pgfqpoint{0.173cm}{1.258cm}}{\pgfqpoint{0.207cm}{1.273cm}}{\pgfqpoint{0.233cm}{1.299cm}}
\pgfpathcurveto{\pgfqpoint{0.259cm}{1.324cm}}{\pgfqpoint{0.273cm}{1.359cm}}{\pgfqpoint{0.273cm}{1.395cm}}
\pgfusepath{fill}
\begin{pgfscope}
\pgfsetdash{}{0cm}
\pgfsetlinewidth{0.818mm}
\pgfsetmiterlimit{7.0}
\pgfpathmoveto{\pgfqpoint{0.682cm}{0.671cm}}
\pgfpathlineto{\pgfqpoint{0.679cm}{1.418cm}}
\pgfusepath{stroke}
\end{pgfscope}
\pgfpathmoveto{\pgfqpoint{0.815cm}{1.399cm}}
\pgfpathcurveto{\pgfqpoint{0.815cm}{1.435cm}}{\pgfqpoint{0.801cm}{1.47cm}}{\pgfqpoint{0.775cm}{1.496cm}}
\pgfpathcurveto{\pgfqpoint{0.75cm}{1.521cm}}{\pgfqpoint{0.715cm}{1.536cm}}{\pgfqpoint{0.679cm}{1.536cm}}
\pgfpathcurveto{\pgfqpoint{0.643cm}{1.536cm}}{\pgfqpoint{0.608cm}{1.521cm}}{\pgfqpoint{0.582cm}{1.496cm}}
\pgfpathcurveto{\pgfqpoint{0.557cm}{1.47cm}}{\pgfqpoint{0.542cm}{1.435cm}}{\pgfqpoint{0.542cm}{1.399cm}}
\pgfpathcurveto{\pgfqpoint{0.542cm}{1.363cm}}{\pgfqpoint{0.557cm}{1.328cm}}{\pgfqpoint{0.582cm}{1.302cm}}
\pgfpathcurveto{\pgfqpoint{0.608cm}{1.276cm}}{\pgfqpoint{0.643cm}{1.262cm}}{\pgfqpoint{0.679cm}{1.262cm}}
\pgfpathcurveto{\pgfqpoint{0.715cm}{1.262cm}}{\pgfqpoint{0.75cm}{1.276cm}}{\pgfqpoint{0.775cm}{1.302cm}}
\pgfpathcurveto{\pgfqpoint{0.801cm}{1.328cm}}{\pgfqpoint{0.815cm}{1.363cm}}{\pgfqpoint{0.815cm}{1.399cm}}
\pgfusepath{fill}
\pgfpathmoveto{\pgfqpoint{1.345cm}{1.371cm}}
\pgfpathcurveto{\pgfqpoint{1.345cm}{1.408cm}}{\pgfqpoint{1.331cm}{1.442cm}}{\pgfqpoint{1.305cm}{1.468cm}}
\pgfpathcurveto{\pgfqpoint{1.28cm}{1.494cm}}{\pgfqpoint{1.245cm}{1.508cm}}{\pgfqpoint{1.209cm}{1.508cm}}
\pgfpathcurveto{\pgfqpoint{1.172cm}{1.508cm}}{\pgfqpoint{1.138cm}{1.494cm}}{\pgfqpoint{1.112cm}{1.468cm}}
\pgfpathcurveto{\pgfqpoint{1.087cm}{1.442cm}}{\pgfqpoint{1.072cm}{1.408cm}}{\pgfqpoint{1.072cm}{1.371cm}}
\pgfpathcurveto{\pgfqpoint{1.072cm}{1.335cm}}{\pgfqpoint{1.087cm}{1.3cm}}{\pgfqpoint{1.112cm}{1.274cm}}
\pgfpathcurveto{\pgfqpoint{1.138cm}{1.249cm}}{\pgfqpoint{1.172cm}{1.234cm}}{\pgfqpoint{1.209cm}{1.234cm}}
\pgfpathcurveto{\pgfqpoint{1.245cm}{1.234cm}}{\pgfqpoint{1.28cm}{1.249cm}}{\pgfqpoint{1.305cm}{1.274cm}}
\pgfpathcurveto{\pgfqpoint{1.331cm}{1.3cm}}{\pgfqpoint{1.345cm}{1.335cm}}{\pgfqpoint{1.345cm}{1.371cm}}
\pgfusepath{fill}
\begin{pgfscope}
\pgfsetdash{}{0cm}
\pgfsetlinewidth{0.818mm}
\pgfsetroundcap
\pgfsetmiterlimit{4.0}
\pgfpathmoveto{\pgfqpoint{0.682cm}{0.671cm}}
\pgfpathlineto{\pgfqpoint{0.682cm}{0.042cm}}
\pgfusepath{stroke}
\end{pgfscope}
\end{pgfscope}
\end{pgfscope}
\end{pgfscope}
\end{tikzpicture}}}$ or
  linear in $\zeta$. Due to Gaussianity of $X$, the first term on the right
  hand side of {\eqref{eq:L24}} vanishes. Our goal is to show that the second
  term behaves like $2^i$ whereas the terms in $R$ are more regular, namely,
  bounded by $2^{i (1 / 2 + \kappa)}$. In other words, $R$ cannot compensate
  $4\lambda L ( X, X, X, X^{\!\resizebox{0.6em}{!}{
\begin{tikzpicture}
\pgfpathmoveto{\pgfqpoint{0cm}{-0.035cm}}
\pgfpathlineto{\pgfqpoint{1.376cm}{-0.035cm}}
\pgfpathlineto{\pgfqpoint{1.376cm}{1.552cm}}
\pgfpathlineto{\pgfqpoint{0cm}{1.552cm}}
\pgfpathclose
\pgfusepath{clip}
\begin{pgfscope}
\begin{pgfscope}
\pgfpathmoveto{\pgfqpoint{0cm}{-0.035cm}}
\pgfpathlineto{\pgfqpoint{1.376cm}{-0.035cm}}
\pgfpathlineto{\pgfqpoint{1.376cm}{1.552cm}}
\pgfpathlineto{\pgfqpoint{0cm}{1.552cm}}
\pgfpathclose
\pgfusepath{clip}
\begin{pgfscope}
\begin{pgfscope}
\pgfsetdash{}{0cm}
\pgfsetlinewidth{0.818mm}
\pgfsetroundcap
\pgfsetroundjoin
\pgfsetmiterlimit{7.0}
\definecolor{eps2pgf_color}{gray}{0}\pgfsetstrokecolor{eps2pgf_color}\pgfsetfillcolor{eps2pgf_color}
\pgfpathmoveto{\pgfqpoint{0.117cm}{1.421cm}}
\pgfpathlineto{\pgfqpoint{0.682cm}{0.671cm}}
\pgfpathlineto{\pgfqpoint{1.246cm}{1.421cm}}
\pgfusepath{stroke}
\end{pgfscope}
\definecolor{eps2pgf_color}{gray}{0}\pgfsetstrokecolor{eps2pgf_color}\pgfsetfillcolor{eps2pgf_color}
\pgfpathmoveto{\pgfqpoint{0.273cm}{1.395cm}}
\pgfpathcurveto{\pgfqpoint{0.273cm}{1.432cm}}{\pgfqpoint{0.259cm}{1.467cm}}{\pgfqpoint{0.233cm}{1.492cm}}
\pgfpathcurveto{\pgfqpoint{0.207cm}{1.518cm}}{\pgfqpoint{0.173cm}{1.532cm}}{\pgfqpoint{0.137cm}{1.532cm}}
\pgfpathcurveto{\pgfqpoint{0.1cm}{1.532cm}}{\pgfqpoint{0.066cm}{1.518cm}}{\pgfqpoint{0.04cm}{1.492cm}}
\pgfpathcurveto{\pgfqpoint{0.014cm}{1.467cm}}{\pgfqpoint{0cm}{1.432cm}}{\pgfqpoint{0cm}{1.395cm}}
\pgfpathcurveto{\pgfqpoint{0cm}{1.359cm}}{\pgfqpoint{0.014cm}{1.324cm}}{\pgfqpoint{0.04cm}{1.299cm}}
\pgfpathcurveto{\pgfqpoint{0.066cm}{1.273cm}}{\pgfqpoint{0.1cm}{1.258cm}}{\pgfqpoint{0.137cm}{1.258cm}}
\pgfpathcurveto{\pgfqpoint{0.173cm}{1.258cm}}{\pgfqpoint{0.207cm}{1.273cm}}{\pgfqpoint{0.233cm}{1.299cm}}
\pgfpathcurveto{\pgfqpoint{0.259cm}{1.324cm}}{\pgfqpoint{0.273cm}{1.359cm}}{\pgfqpoint{0.273cm}{1.395cm}}
\pgfusepath{fill}
\begin{pgfscope}
\pgfsetdash{}{0cm}
\pgfsetlinewidth{0.818mm}
\pgfsetmiterlimit{7.0}
\pgfpathmoveto{\pgfqpoint{0.682cm}{0.671cm}}
\pgfpathlineto{\pgfqpoint{0.679cm}{1.418cm}}
\pgfusepath{stroke}
\end{pgfscope}
\pgfpathmoveto{\pgfqpoint{0.815cm}{1.399cm}}
\pgfpathcurveto{\pgfqpoint{0.815cm}{1.435cm}}{\pgfqpoint{0.801cm}{1.47cm}}{\pgfqpoint{0.775cm}{1.496cm}}
\pgfpathcurveto{\pgfqpoint{0.75cm}{1.521cm}}{\pgfqpoint{0.715cm}{1.536cm}}{\pgfqpoint{0.679cm}{1.536cm}}
\pgfpathcurveto{\pgfqpoint{0.643cm}{1.536cm}}{\pgfqpoint{0.608cm}{1.521cm}}{\pgfqpoint{0.582cm}{1.496cm}}
\pgfpathcurveto{\pgfqpoint{0.557cm}{1.47cm}}{\pgfqpoint{0.542cm}{1.435cm}}{\pgfqpoint{0.542cm}{1.399cm}}
\pgfpathcurveto{\pgfqpoint{0.542cm}{1.363cm}}{\pgfqpoint{0.557cm}{1.328cm}}{\pgfqpoint{0.582cm}{1.302cm}}
\pgfpathcurveto{\pgfqpoint{0.608cm}{1.276cm}}{\pgfqpoint{0.643cm}{1.262cm}}{\pgfqpoint{0.679cm}{1.262cm}}
\pgfpathcurveto{\pgfqpoint{0.715cm}{1.262cm}}{\pgfqpoint{0.75cm}{1.276cm}}{\pgfqpoint{0.775cm}{1.302cm}}
\pgfpathcurveto{\pgfqpoint{0.801cm}{1.328cm}}{\pgfqpoint{0.815cm}{1.363cm}}{\pgfqpoint{0.815cm}{1.399cm}}
\pgfusepath{fill}
\pgfpathmoveto{\pgfqpoint{1.345cm}{1.371cm}}
\pgfpathcurveto{\pgfqpoint{1.345cm}{1.408cm}}{\pgfqpoint{1.331cm}{1.442cm}}{\pgfqpoint{1.305cm}{1.468cm}}
\pgfpathcurveto{\pgfqpoint{1.28cm}{1.494cm}}{\pgfqpoint{1.245cm}{1.508cm}}{\pgfqpoint{1.209cm}{1.508cm}}
\pgfpathcurveto{\pgfqpoint{1.172cm}{1.508cm}}{\pgfqpoint{1.138cm}{1.494cm}}{\pgfqpoint{1.112cm}{1.468cm}}
\pgfpathcurveto{\pgfqpoint{1.087cm}{1.442cm}}{\pgfqpoint{1.072cm}{1.408cm}}{\pgfqpoint{1.072cm}{1.371cm}}
\pgfpathcurveto{\pgfqpoint{1.072cm}{1.335cm}}{\pgfqpoint{1.087cm}{1.3cm}}{\pgfqpoint{1.112cm}{1.274cm}}
\pgfpathcurveto{\pgfqpoint{1.138cm}{1.249cm}}{\pgfqpoint{1.172cm}{1.234cm}}{\pgfqpoint{1.209cm}{1.234cm}}
\pgfpathcurveto{\pgfqpoint{1.245cm}{1.234cm}}{\pgfqpoint{1.28cm}{1.249cm}}{\pgfqpoint{1.305cm}{1.274cm}}
\pgfpathcurveto{\pgfqpoint{1.331cm}{1.3cm}}{\pgfqpoint{1.345cm}{1.335cm}}{\pgfqpoint{1.345cm}{1.371cm}}
\pgfusepath{fill}
\begin{pgfscope}
\pgfsetdash{}{0cm}
\pgfsetlinewidth{0.818mm}
\pgfsetroundcap
\pgfsetmiterlimit{4.0}
\pgfpathmoveto{\pgfqpoint{0.682cm}{0.671cm}}
\pgfpathlineto{\pgfqpoint{0.682cm}{0.042cm}}
\pgfusepath{stroke}
\end{pgfscope}
\end{pgfscope}
\end{pgfscope}
\end{pgfscope}
\end{tikzpicture}}} )$ and as a consequence $L (\varphi,
  \varphi, \varphi, \varphi) \neq 0$ if $\lambda> 0$.
  
  Let us begin with $L ( X, X, X, X^{\!\resizebox{0.6em}{!}{
\begin{tikzpicture}
\pgfpathmoveto{\pgfqpoint{0cm}{-0.035cm}}
\pgfpathlineto{\pgfqpoint{1.376cm}{-0.035cm}}
\pgfpathlineto{\pgfqpoint{1.376cm}{1.552cm}}
\pgfpathlineto{\pgfqpoint{0cm}{1.552cm}}
\pgfpathclose
\pgfusepath{clip}
\begin{pgfscope}
\begin{pgfscope}
\pgfpathmoveto{\pgfqpoint{0cm}{-0.035cm}}
\pgfpathlineto{\pgfqpoint{1.376cm}{-0.035cm}}
\pgfpathlineto{\pgfqpoint{1.376cm}{1.552cm}}
\pgfpathlineto{\pgfqpoint{0cm}{1.552cm}}
\pgfpathclose
\pgfusepath{clip}
\begin{pgfscope}
\begin{pgfscope}
\pgfsetdash{}{0cm}
\pgfsetlinewidth{0.818mm}
\pgfsetroundcap
\pgfsetroundjoin
\pgfsetmiterlimit{7.0}
\definecolor{eps2pgf_color}{gray}{0}\pgfsetstrokecolor{eps2pgf_color}\pgfsetfillcolor{eps2pgf_color}
\pgfpathmoveto{\pgfqpoint{0.117cm}{1.421cm}}
\pgfpathlineto{\pgfqpoint{0.682cm}{0.671cm}}
\pgfpathlineto{\pgfqpoint{1.246cm}{1.421cm}}
\pgfusepath{stroke}
\end{pgfscope}
\definecolor{eps2pgf_color}{gray}{0}\pgfsetstrokecolor{eps2pgf_color}\pgfsetfillcolor{eps2pgf_color}
\pgfpathmoveto{\pgfqpoint{0.273cm}{1.395cm}}
\pgfpathcurveto{\pgfqpoint{0.273cm}{1.432cm}}{\pgfqpoint{0.259cm}{1.467cm}}{\pgfqpoint{0.233cm}{1.492cm}}
\pgfpathcurveto{\pgfqpoint{0.207cm}{1.518cm}}{\pgfqpoint{0.173cm}{1.532cm}}{\pgfqpoint{0.137cm}{1.532cm}}
\pgfpathcurveto{\pgfqpoint{0.1cm}{1.532cm}}{\pgfqpoint{0.066cm}{1.518cm}}{\pgfqpoint{0.04cm}{1.492cm}}
\pgfpathcurveto{\pgfqpoint{0.014cm}{1.467cm}}{\pgfqpoint{0cm}{1.432cm}}{\pgfqpoint{0cm}{1.395cm}}
\pgfpathcurveto{\pgfqpoint{0cm}{1.359cm}}{\pgfqpoint{0.014cm}{1.324cm}}{\pgfqpoint{0.04cm}{1.299cm}}
\pgfpathcurveto{\pgfqpoint{0.066cm}{1.273cm}}{\pgfqpoint{0.1cm}{1.258cm}}{\pgfqpoint{0.137cm}{1.258cm}}
\pgfpathcurveto{\pgfqpoint{0.173cm}{1.258cm}}{\pgfqpoint{0.207cm}{1.273cm}}{\pgfqpoint{0.233cm}{1.299cm}}
\pgfpathcurveto{\pgfqpoint{0.259cm}{1.324cm}}{\pgfqpoint{0.273cm}{1.359cm}}{\pgfqpoint{0.273cm}{1.395cm}}
\pgfusepath{fill}
\begin{pgfscope}
\pgfsetdash{}{0cm}
\pgfsetlinewidth{0.818mm}
\pgfsetmiterlimit{7.0}
\pgfpathmoveto{\pgfqpoint{0.682cm}{0.671cm}}
\pgfpathlineto{\pgfqpoint{0.679cm}{1.418cm}}
\pgfusepath{stroke}
\end{pgfscope}
\pgfpathmoveto{\pgfqpoint{0.815cm}{1.399cm}}
\pgfpathcurveto{\pgfqpoint{0.815cm}{1.435cm}}{\pgfqpoint{0.801cm}{1.47cm}}{\pgfqpoint{0.775cm}{1.496cm}}
\pgfpathcurveto{\pgfqpoint{0.75cm}{1.521cm}}{\pgfqpoint{0.715cm}{1.536cm}}{\pgfqpoint{0.679cm}{1.536cm}}
\pgfpathcurveto{\pgfqpoint{0.643cm}{1.536cm}}{\pgfqpoint{0.608cm}{1.521cm}}{\pgfqpoint{0.582cm}{1.496cm}}
\pgfpathcurveto{\pgfqpoint{0.557cm}{1.47cm}}{\pgfqpoint{0.542cm}{1.435cm}}{\pgfqpoint{0.542cm}{1.399cm}}
\pgfpathcurveto{\pgfqpoint{0.542cm}{1.363cm}}{\pgfqpoint{0.557cm}{1.328cm}}{\pgfqpoint{0.582cm}{1.302cm}}
\pgfpathcurveto{\pgfqpoint{0.608cm}{1.276cm}}{\pgfqpoint{0.643cm}{1.262cm}}{\pgfqpoint{0.679cm}{1.262cm}}
\pgfpathcurveto{\pgfqpoint{0.715cm}{1.262cm}}{\pgfqpoint{0.75cm}{1.276cm}}{\pgfqpoint{0.775cm}{1.302cm}}
\pgfpathcurveto{\pgfqpoint{0.801cm}{1.328cm}}{\pgfqpoint{0.815cm}{1.363cm}}{\pgfqpoint{0.815cm}{1.399cm}}
\pgfusepath{fill}
\pgfpathmoveto{\pgfqpoint{1.345cm}{1.371cm}}
\pgfpathcurveto{\pgfqpoint{1.345cm}{1.408cm}}{\pgfqpoint{1.331cm}{1.442cm}}{\pgfqpoint{1.305cm}{1.468cm}}
\pgfpathcurveto{\pgfqpoint{1.28cm}{1.494cm}}{\pgfqpoint{1.245cm}{1.508cm}}{\pgfqpoint{1.209cm}{1.508cm}}
\pgfpathcurveto{\pgfqpoint{1.172cm}{1.508cm}}{\pgfqpoint{1.138cm}{1.494cm}}{\pgfqpoint{1.112cm}{1.468cm}}
\pgfpathcurveto{\pgfqpoint{1.087cm}{1.442cm}}{\pgfqpoint{1.072cm}{1.408cm}}{\pgfqpoint{1.072cm}{1.371cm}}
\pgfpathcurveto{\pgfqpoint{1.072cm}{1.335cm}}{\pgfqpoint{1.087cm}{1.3cm}}{\pgfqpoint{1.112cm}{1.274cm}}
\pgfpathcurveto{\pgfqpoint{1.138cm}{1.249cm}}{\pgfqpoint{1.172cm}{1.234cm}}{\pgfqpoint{1.209cm}{1.234cm}}
\pgfpathcurveto{\pgfqpoint{1.245cm}{1.234cm}}{\pgfqpoint{1.28cm}{1.249cm}}{\pgfqpoint{1.305cm}{1.274cm}}
\pgfpathcurveto{\pgfqpoint{1.331cm}{1.3cm}}{\pgfqpoint{1.345cm}{1.335cm}}{\pgfqpoint{1.345cm}{1.371cm}}
\pgfusepath{fill}
\begin{pgfscope}
\pgfsetdash{}{0cm}
\pgfsetlinewidth{0.818mm}
\pgfsetroundcap
\pgfsetmiterlimit{4.0}
\pgfpathmoveto{\pgfqpoint{0.682cm}{0.671cm}}
\pgfpathlineto{\pgfqpoint{0.682cm}{0.042cm}}
\pgfusepath{stroke}
\end{pgfscope}
\end{pgfscope}
\end{pgfscope}
\end{pgfscope}
\end{tikzpicture}}} )$. To this end, we denote $k_{[123]}=k_{1}+k_{2}+k_{3}$ and
  recall that
  \[ (\Delta_i X) (0) = \int_{\mathbb{R}^d} \varphi_i (k) \int_{- \infty}^0
     e^{- [m^{2} + | k |^2] (- s)} \hat{\xi} (\mathd s, \mathd k), \]
  \[ ( \Delta_i X^{\!\resizebox{0.6em}{!}{
\begin{tikzpicture}
\pgfpathmoveto{\pgfqpoint{0cm}{-0.035cm}}
\pgfpathlineto{\pgfqpoint{1.376cm}{-0.035cm}}
\pgfpathlineto{\pgfqpoint{1.376cm}{1.552cm}}
\pgfpathlineto{\pgfqpoint{0cm}{1.552cm}}
\pgfpathclose
\pgfusepath{clip}
\begin{pgfscope}
\begin{pgfscope}
\pgfpathmoveto{\pgfqpoint{0cm}{-0.035cm}}
\pgfpathlineto{\pgfqpoint{1.376cm}{-0.035cm}}
\pgfpathlineto{\pgfqpoint{1.376cm}{1.552cm}}
\pgfpathlineto{\pgfqpoint{0cm}{1.552cm}}
\pgfpathclose
\pgfusepath{clip}
\begin{pgfscope}
\begin{pgfscope}
\pgfsetdash{}{0cm}
\pgfsetlinewidth{0.818mm}
\pgfsetroundcap
\pgfsetroundjoin
\pgfsetmiterlimit{7.0}
\definecolor{eps2pgf_color}{gray}{0}\pgfsetstrokecolor{eps2pgf_color}\pgfsetfillcolor{eps2pgf_color}
\pgfpathmoveto{\pgfqpoint{0.117cm}{1.421cm}}
\pgfpathlineto{\pgfqpoint{0.682cm}{0.671cm}}
\pgfpathlineto{\pgfqpoint{1.246cm}{1.421cm}}
\pgfusepath{stroke}
\end{pgfscope}
\definecolor{eps2pgf_color}{gray}{0}\pgfsetstrokecolor{eps2pgf_color}\pgfsetfillcolor{eps2pgf_color}
\pgfpathmoveto{\pgfqpoint{0.273cm}{1.395cm}}
\pgfpathcurveto{\pgfqpoint{0.273cm}{1.432cm}}{\pgfqpoint{0.259cm}{1.467cm}}{\pgfqpoint{0.233cm}{1.492cm}}
\pgfpathcurveto{\pgfqpoint{0.207cm}{1.518cm}}{\pgfqpoint{0.173cm}{1.532cm}}{\pgfqpoint{0.137cm}{1.532cm}}
\pgfpathcurveto{\pgfqpoint{0.1cm}{1.532cm}}{\pgfqpoint{0.066cm}{1.518cm}}{\pgfqpoint{0.04cm}{1.492cm}}
\pgfpathcurveto{\pgfqpoint{0.014cm}{1.467cm}}{\pgfqpoint{0cm}{1.432cm}}{\pgfqpoint{0cm}{1.395cm}}
\pgfpathcurveto{\pgfqpoint{0cm}{1.359cm}}{\pgfqpoint{0.014cm}{1.324cm}}{\pgfqpoint{0.04cm}{1.299cm}}
\pgfpathcurveto{\pgfqpoint{0.066cm}{1.273cm}}{\pgfqpoint{0.1cm}{1.258cm}}{\pgfqpoint{0.137cm}{1.258cm}}
\pgfpathcurveto{\pgfqpoint{0.173cm}{1.258cm}}{\pgfqpoint{0.207cm}{1.273cm}}{\pgfqpoint{0.233cm}{1.299cm}}
\pgfpathcurveto{\pgfqpoint{0.259cm}{1.324cm}}{\pgfqpoint{0.273cm}{1.359cm}}{\pgfqpoint{0.273cm}{1.395cm}}
\pgfusepath{fill}
\begin{pgfscope}
\pgfsetdash{}{0cm}
\pgfsetlinewidth{0.818mm}
\pgfsetmiterlimit{7.0}
\pgfpathmoveto{\pgfqpoint{0.682cm}{0.671cm}}
\pgfpathlineto{\pgfqpoint{0.679cm}{1.418cm}}
\pgfusepath{stroke}
\end{pgfscope}
\pgfpathmoveto{\pgfqpoint{0.815cm}{1.399cm}}
\pgfpathcurveto{\pgfqpoint{0.815cm}{1.435cm}}{\pgfqpoint{0.801cm}{1.47cm}}{\pgfqpoint{0.775cm}{1.496cm}}
\pgfpathcurveto{\pgfqpoint{0.75cm}{1.521cm}}{\pgfqpoint{0.715cm}{1.536cm}}{\pgfqpoint{0.679cm}{1.536cm}}
\pgfpathcurveto{\pgfqpoint{0.643cm}{1.536cm}}{\pgfqpoint{0.608cm}{1.521cm}}{\pgfqpoint{0.582cm}{1.496cm}}
\pgfpathcurveto{\pgfqpoint{0.557cm}{1.47cm}}{\pgfqpoint{0.542cm}{1.435cm}}{\pgfqpoint{0.542cm}{1.399cm}}
\pgfpathcurveto{\pgfqpoint{0.542cm}{1.363cm}}{\pgfqpoint{0.557cm}{1.328cm}}{\pgfqpoint{0.582cm}{1.302cm}}
\pgfpathcurveto{\pgfqpoint{0.608cm}{1.276cm}}{\pgfqpoint{0.643cm}{1.262cm}}{\pgfqpoint{0.679cm}{1.262cm}}
\pgfpathcurveto{\pgfqpoint{0.715cm}{1.262cm}}{\pgfqpoint{0.75cm}{1.276cm}}{\pgfqpoint{0.775cm}{1.302cm}}
\pgfpathcurveto{\pgfqpoint{0.801cm}{1.328cm}}{\pgfqpoint{0.815cm}{1.363cm}}{\pgfqpoint{0.815cm}{1.399cm}}
\pgfusepath{fill}
\pgfpathmoveto{\pgfqpoint{1.345cm}{1.371cm}}
\pgfpathcurveto{\pgfqpoint{1.345cm}{1.408cm}}{\pgfqpoint{1.331cm}{1.442cm}}{\pgfqpoint{1.305cm}{1.468cm}}
\pgfpathcurveto{\pgfqpoint{1.28cm}{1.494cm}}{\pgfqpoint{1.245cm}{1.508cm}}{\pgfqpoint{1.209cm}{1.508cm}}
\pgfpathcurveto{\pgfqpoint{1.172cm}{1.508cm}}{\pgfqpoint{1.138cm}{1.494cm}}{\pgfqpoint{1.112cm}{1.468cm}}
\pgfpathcurveto{\pgfqpoint{1.087cm}{1.442cm}}{\pgfqpoint{1.072cm}{1.408cm}}{\pgfqpoint{1.072cm}{1.371cm}}
\pgfpathcurveto{\pgfqpoint{1.072cm}{1.335cm}}{\pgfqpoint{1.087cm}{1.3cm}}{\pgfqpoint{1.112cm}{1.274cm}}
\pgfpathcurveto{\pgfqpoint{1.138cm}{1.249cm}}{\pgfqpoint{1.172cm}{1.234cm}}{\pgfqpoint{1.209cm}{1.234cm}}
\pgfpathcurveto{\pgfqpoint{1.245cm}{1.234cm}}{\pgfqpoint{1.28cm}{1.249cm}}{\pgfqpoint{1.305cm}{1.274cm}}
\pgfpathcurveto{\pgfqpoint{1.331cm}{1.3cm}}{\pgfqpoint{1.345cm}{1.335cm}}{\pgfqpoint{1.345cm}{1.371cm}}
\pgfusepath{fill}
\begin{pgfscope}
\pgfsetdash{}{0cm}
\pgfsetlinewidth{0.818mm}
\pgfsetroundcap
\pgfsetmiterlimit{4.0}
\pgfpathmoveto{\pgfqpoint{0.682cm}{0.671cm}}
\pgfpathlineto{\pgfqpoint{0.682cm}{0.042cm}}
\pgfusepath{stroke}
\end{pgfscope}
\end{pgfscope}
\end{pgfscope}
\end{pgfscope}
\end{tikzpicture}}} ) (0) = \int^0_{- \infty} \mathd s
     \int_{\mathbb{R}^d} \int_{\mathbb{R}^d} \int_{\mathbb{R}^d} \varphi_i
     (k_{[123]}) e^{- [m^{2} + | k_{[123]} |^2] (- s)} \]
  \[ \times \left\llbracket \prod_{l = 1, 2, 3} \int^s_{- \infty} e^{- [m^{2}+ |
     k_l |^2] (s - s_l)} \hat{\xi} (\mathd s_l, \mathd k_l) \right\rrbracket,
  \]
  where $\llbracket\cdot\rrbracket$ denotes Wick's product.
Hence denoting $H\assign 
     [4m^{2} + | k_{[123]} |^2+|k_{1}|^{2}+|k_{2}|^{2}+|k_{3}|^{2}] $ we obtain
  \[ L ( X, X, X, X^{\!\resizebox{0.6em}{!}{
\begin{tikzpicture}
\pgfpathmoveto{\pgfqpoint{0cm}{-0.035cm}}
\pgfpathlineto{\pgfqpoint{1.376cm}{-0.035cm}}
\pgfpathlineto{\pgfqpoint{1.376cm}{1.552cm}}
\pgfpathlineto{\pgfqpoint{0cm}{1.552cm}}
\pgfpathclose
\pgfusepath{clip}
\begin{pgfscope}
\begin{pgfscope}
\pgfpathmoveto{\pgfqpoint{0cm}{-0.035cm}}
\pgfpathlineto{\pgfqpoint{1.376cm}{-0.035cm}}
\pgfpathlineto{\pgfqpoint{1.376cm}{1.552cm}}
\pgfpathlineto{\pgfqpoint{0cm}{1.552cm}}
\pgfpathclose
\pgfusepath{clip}
\begin{pgfscope}
\begin{pgfscope}
\pgfsetdash{}{0cm}
\pgfsetlinewidth{0.818mm}
\pgfsetroundcap
\pgfsetroundjoin
\pgfsetmiterlimit{7.0}
\definecolor{eps2pgf_color}{gray}{0}\pgfsetstrokecolor{eps2pgf_color}\pgfsetfillcolor{eps2pgf_color}
\pgfpathmoveto{\pgfqpoint{0.117cm}{1.421cm}}
\pgfpathlineto{\pgfqpoint{0.682cm}{0.671cm}}
\pgfpathlineto{\pgfqpoint{1.246cm}{1.421cm}}
\pgfusepath{stroke}
\end{pgfscope}
\definecolor{eps2pgf_color}{gray}{0}\pgfsetstrokecolor{eps2pgf_color}\pgfsetfillcolor{eps2pgf_color}
\pgfpathmoveto{\pgfqpoint{0.273cm}{1.395cm}}
\pgfpathcurveto{\pgfqpoint{0.273cm}{1.432cm}}{\pgfqpoint{0.259cm}{1.467cm}}{\pgfqpoint{0.233cm}{1.492cm}}
\pgfpathcurveto{\pgfqpoint{0.207cm}{1.518cm}}{\pgfqpoint{0.173cm}{1.532cm}}{\pgfqpoint{0.137cm}{1.532cm}}
\pgfpathcurveto{\pgfqpoint{0.1cm}{1.532cm}}{\pgfqpoint{0.066cm}{1.518cm}}{\pgfqpoint{0.04cm}{1.492cm}}
\pgfpathcurveto{\pgfqpoint{0.014cm}{1.467cm}}{\pgfqpoint{0cm}{1.432cm}}{\pgfqpoint{0cm}{1.395cm}}
\pgfpathcurveto{\pgfqpoint{0cm}{1.359cm}}{\pgfqpoint{0.014cm}{1.324cm}}{\pgfqpoint{0.04cm}{1.299cm}}
\pgfpathcurveto{\pgfqpoint{0.066cm}{1.273cm}}{\pgfqpoint{0.1cm}{1.258cm}}{\pgfqpoint{0.137cm}{1.258cm}}
\pgfpathcurveto{\pgfqpoint{0.173cm}{1.258cm}}{\pgfqpoint{0.207cm}{1.273cm}}{\pgfqpoint{0.233cm}{1.299cm}}
\pgfpathcurveto{\pgfqpoint{0.259cm}{1.324cm}}{\pgfqpoint{0.273cm}{1.359cm}}{\pgfqpoint{0.273cm}{1.395cm}}
\pgfusepath{fill}
\begin{pgfscope}
\pgfsetdash{}{0cm}
\pgfsetlinewidth{0.818mm}
\pgfsetmiterlimit{7.0}
\pgfpathmoveto{\pgfqpoint{0.682cm}{0.671cm}}
\pgfpathlineto{\pgfqpoint{0.679cm}{1.418cm}}
\pgfusepath{stroke}
\end{pgfscope}
\pgfpathmoveto{\pgfqpoint{0.815cm}{1.399cm}}
\pgfpathcurveto{\pgfqpoint{0.815cm}{1.435cm}}{\pgfqpoint{0.801cm}{1.47cm}}{\pgfqpoint{0.775cm}{1.496cm}}
\pgfpathcurveto{\pgfqpoint{0.75cm}{1.521cm}}{\pgfqpoint{0.715cm}{1.536cm}}{\pgfqpoint{0.679cm}{1.536cm}}
\pgfpathcurveto{\pgfqpoint{0.643cm}{1.536cm}}{\pgfqpoint{0.608cm}{1.521cm}}{\pgfqpoint{0.582cm}{1.496cm}}
\pgfpathcurveto{\pgfqpoint{0.557cm}{1.47cm}}{\pgfqpoint{0.542cm}{1.435cm}}{\pgfqpoint{0.542cm}{1.399cm}}
\pgfpathcurveto{\pgfqpoint{0.542cm}{1.363cm}}{\pgfqpoint{0.557cm}{1.328cm}}{\pgfqpoint{0.582cm}{1.302cm}}
\pgfpathcurveto{\pgfqpoint{0.608cm}{1.276cm}}{\pgfqpoint{0.643cm}{1.262cm}}{\pgfqpoint{0.679cm}{1.262cm}}
\pgfpathcurveto{\pgfqpoint{0.715cm}{1.262cm}}{\pgfqpoint{0.75cm}{1.276cm}}{\pgfqpoint{0.775cm}{1.302cm}}
\pgfpathcurveto{\pgfqpoint{0.801cm}{1.328cm}}{\pgfqpoint{0.815cm}{1.363cm}}{\pgfqpoint{0.815cm}{1.399cm}}
\pgfusepath{fill}
\pgfpathmoveto{\pgfqpoint{1.345cm}{1.371cm}}
\pgfpathcurveto{\pgfqpoint{1.345cm}{1.408cm}}{\pgfqpoint{1.331cm}{1.442cm}}{\pgfqpoint{1.305cm}{1.468cm}}
\pgfpathcurveto{\pgfqpoint{1.28cm}{1.494cm}}{\pgfqpoint{1.245cm}{1.508cm}}{\pgfqpoint{1.209cm}{1.508cm}}
\pgfpathcurveto{\pgfqpoint{1.172cm}{1.508cm}}{\pgfqpoint{1.138cm}{1.494cm}}{\pgfqpoint{1.112cm}{1.468cm}}
\pgfpathcurveto{\pgfqpoint{1.087cm}{1.442cm}}{\pgfqpoint{1.072cm}{1.408cm}}{\pgfqpoint{1.072cm}{1.371cm}}
\pgfpathcurveto{\pgfqpoint{1.072cm}{1.335cm}}{\pgfqpoint{1.087cm}{1.3cm}}{\pgfqpoint{1.112cm}{1.274cm}}
\pgfpathcurveto{\pgfqpoint{1.138cm}{1.249cm}}{\pgfqpoint{1.172cm}{1.234cm}}{\pgfqpoint{1.209cm}{1.234cm}}
\pgfpathcurveto{\pgfqpoint{1.245cm}{1.234cm}}{\pgfqpoint{1.28cm}{1.249cm}}{\pgfqpoint{1.305cm}{1.274cm}}
\pgfpathcurveto{\pgfqpoint{1.331cm}{1.3cm}}{\pgfqpoint{1.345cm}{1.335cm}}{\pgfqpoint{1.345cm}{1.371cm}}
\pgfusepath{fill}
\begin{pgfscope}
\pgfsetdash{}{0cm}
\pgfsetlinewidth{0.818mm}
\pgfsetroundcap
\pgfsetmiterlimit{4.0}
\pgfpathmoveto{\pgfqpoint{0.682cm}{0.671cm}}
\pgfpathlineto{\pgfqpoint{0.682cm}{0.042cm}}
\pgfusepath{stroke}
\end{pgfscope}
\end{pgfscope}
\end{pgfscope}
\end{pgfscope}
\end{tikzpicture}}} ) =\mathbb{E} \left[ (\Delta_i X)
     (0) (\Delta_i X) (0) (\Delta_i X) (0) ( \Delta_i X^{\!\resizebox{0.6em}{!}{
\begin{tikzpicture}
\pgfpathmoveto{\pgfqpoint{0cm}{-0.035cm}}
\pgfpathlineto{\pgfqpoint{1.376cm}{-0.035cm}}
\pgfpathlineto{\pgfqpoint{1.376cm}{1.552cm}}
\pgfpathlineto{\pgfqpoint{0cm}{1.552cm}}
\pgfpathclose
\pgfusepath{clip}
\begin{pgfscope}
\begin{pgfscope}
\pgfpathmoveto{\pgfqpoint{0cm}{-0.035cm}}
\pgfpathlineto{\pgfqpoint{1.376cm}{-0.035cm}}
\pgfpathlineto{\pgfqpoint{1.376cm}{1.552cm}}
\pgfpathlineto{\pgfqpoint{0cm}{1.552cm}}
\pgfpathclose
\pgfusepath{clip}
\begin{pgfscope}
\begin{pgfscope}
\pgfsetdash{}{0cm}
\pgfsetlinewidth{0.818mm}
\pgfsetroundcap
\pgfsetroundjoin
\pgfsetmiterlimit{7.0}
\definecolor{eps2pgf_color}{gray}{0}\pgfsetstrokecolor{eps2pgf_color}\pgfsetfillcolor{eps2pgf_color}
\pgfpathmoveto{\pgfqpoint{0.117cm}{1.421cm}}
\pgfpathlineto{\pgfqpoint{0.682cm}{0.671cm}}
\pgfpathlineto{\pgfqpoint{1.246cm}{1.421cm}}
\pgfusepath{stroke}
\end{pgfscope}
\definecolor{eps2pgf_color}{gray}{0}\pgfsetstrokecolor{eps2pgf_color}\pgfsetfillcolor{eps2pgf_color}
\pgfpathmoveto{\pgfqpoint{0.273cm}{1.395cm}}
\pgfpathcurveto{\pgfqpoint{0.273cm}{1.432cm}}{\pgfqpoint{0.259cm}{1.467cm}}{\pgfqpoint{0.233cm}{1.492cm}}
\pgfpathcurveto{\pgfqpoint{0.207cm}{1.518cm}}{\pgfqpoint{0.173cm}{1.532cm}}{\pgfqpoint{0.137cm}{1.532cm}}
\pgfpathcurveto{\pgfqpoint{0.1cm}{1.532cm}}{\pgfqpoint{0.066cm}{1.518cm}}{\pgfqpoint{0.04cm}{1.492cm}}
\pgfpathcurveto{\pgfqpoint{0.014cm}{1.467cm}}{\pgfqpoint{0cm}{1.432cm}}{\pgfqpoint{0cm}{1.395cm}}
\pgfpathcurveto{\pgfqpoint{0cm}{1.359cm}}{\pgfqpoint{0.014cm}{1.324cm}}{\pgfqpoint{0.04cm}{1.299cm}}
\pgfpathcurveto{\pgfqpoint{0.066cm}{1.273cm}}{\pgfqpoint{0.1cm}{1.258cm}}{\pgfqpoint{0.137cm}{1.258cm}}
\pgfpathcurveto{\pgfqpoint{0.173cm}{1.258cm}}{\pgfqpoint{0.207cm}{1.273cm}}{\pgfqpoint{0.233cm}{1.299cm}}
\pgfpathcurveto{\pgfqpoint{0.259cm}{1.324cm}}{\pgfqpoint{0.273cm}{1.359cm}}{\pgfqpoint{0.273cm}{1.395cm}}
\pgfusepath{fill}
\begin{pgfscope}
\pgfsetdash{}{0cm}
\pgfsetlinewidth{0.818mm}
\pgfsetmiterlimit{7.0}
\pgfpathmoveto{\pgfqpoint{0.682cm}{0.671cm}}
\pgfpathlineto{\pgfqpoint{0.679cm}{1.418cm}}
\pgfusepath{stroke}
\end{pgfscope}
\pgfpathmoveto{\pgfqpoint{0.815cm}{1.399cm}}
\pgfpathcurveto{\pgfqpoint{0.815cm}{1.435cm}}{\pgfqpoint{0.801cm}{1.47cm}}{\pgfqpoint{0.775cm}{1.496cm}}
\pgfpathcurveto{\pgfqpoint{0.75cm}{1.521cm}}{\pgfqpoint{0.715cm}{1.536cm}}{\pgfqpoint{0.679cm}{1.536cm}}
\pgfpathcurveto{\pgfqpoint{0.643cm}{1.536cm}}{\pgfqpoint{0.608cm}{1.521cm}}{\pgfqpoint{0.582cm}{1.496cm}}
\pgfpathcurveto{\pgfqpoint{0.557cm}{1.47cm}}{\pgfqpoint{0.542cm}{1.435cm}}{\pgfqpoint{0.542cm}{1.399cm}}
\pgfpathcurveto{\pgfqpoint{0.542cm}{1.363cm}}{\pgfqpoint{0.557cm}{1.328cm}}{\pgfqpoint{0.582cm}{1.302cm}}
\pgfpathcurveto{\pgfqpoint{0.608cm}{1.276cm}}{\pgfqpoint{0.643cm}{1.262cm}}{\pgfqpoint{0.679cm}{1.262cm}}
\pgfpathcurveto{\pgfqpoint{0.715cm}{1.262cm}}{\pgfqpoint{0.75cm}{1.276cm}}{\pgfqpoint{0.775cm}{1.302cm}}
\pgfpathcurveto{\pgfqpoint{0.801cm}{1.328cm}}{\pgfqpoint{0.815cm}{1.363cm}}{\pgfqpoint{0.815cm}{1.399cm}}
\pgfusepath{fill}
\pgfpathmoveto{\pgfqpoint{1.345cm}{1.371cm}}
\pgfpathcurveto{\pgfqpoint{1.345cm}{1.408cm}}{\pgfqpoint{1.331cm}{1.442cm}}{\pgfqpoint{1.305cm}{1.468cm}}
\pgfpathcurveto{\pgfqpoint{1.28cm}{1.494cm}}{\pgfqpoint{1.245cm}{1.508cm}}{\pgfqpoint{1.209cm}{1.508cm}}
\pgfpathcurveto{\pgfqpoint{1.172cm}{1.508cm}}{\pgfqpoint{1.138cm}{1.494cm}}{\pgfqpoint{1.112cm}{1.468cm}}
\pgfpathcurveto{\pgfqpoint{1.087cm}{1.442cm}}{\pgfqpoint{1.072cm}{1.408cm}}{\pgfqpoint{1.072cm}{1.371cm}}
\pgfpathcurveto{\pgfqpoint{1.072cm}{1.335cm}}{\pgfqpoint{1.087cm}{1.3cm}}{\pgfqpoint{1.112cm}{1.274cm}}
\pgfpathcurveto{\pgfqpoint{1.138cm}{1.249cm}}{\pgfqpoint{1.172cm}{1.234cm}}{\pgfqpoint{1.209cm}{1.234cm}}
\pgfpathcurveto{\pgfqpoint{1.245cm}{1.234cm}}{\pgfqpoint{1.28cm}{1.249cm}}{\pgfqpoint{1.305cm}{1.274cm}}
\pgfpathcurveto{\pgfqpoint{1.331cm}{1.3cm}}{\pgfqpoint{1.345cm}{1.335cm}}{\pgfqpoint{1.345cm}{1.371cm}}
\pgfusepath{fill}
\begin{pgfscope}
\pgfsetdash{}{0cm}
\pgfsetlinewidth{0.818mm}
\pgfsetroundcap
\pgfsetmiterlimit{4.0}
\pgfpathmoveto{\pgfqpoint{0.682cm}{0.671cm}}
\pgfpathlineto{\pgfqpoint{0.682cm}{0.042cm}}
\pgfusepath{stroke}
\end{pgfscope}
\end{pgfscope}
\end{pgfscope}
\end{pgfscope}
\end{tikzpicture}}}
     ) (0) \right] \]
  \[ = 3! \int^0_{- \infty} \mathd s \int_{\mathbb{R}^d}
     \int_{\mathbb{R}^d} \int_{\mathbb{R}^d} \varphi_i (k_{[123]})  e^{-H(-s)}  \prod_{l = 1, 2, 3} \left[ \int_{- \infty}^s e^{- 
     2[m^{2}+| k_l |^2] (s - s_l)} \varphi_i (k_l) \mathd s_l \mathd k_l \right]
  \]
  \[ = \frac{3!}{8} \int^0_{- \infty} \mathd s
     \int_{\mathbb{R}^d} \int_{\mathbb{R}^d} \int_{\mathbb{R}^d} \varphi_i
     (k_{[123]}) e^{-H(-s)} \prod_{l = 1,
     2, 3} \left[ \varphi_i (k_l) \frac{\mathd k_l}{m^{2} + | k_l |^2} \right] \]
  \[ \  \]
  \[ = \frac{3!}{8} \int_{\mathbb{R}^d} \int_{\mathbb{R}^d}
     \int_{\mathbb{R}^d} \frac{\varphi_i (k_{[123]})}{H} \prod_{l = 1, 2, 3} \left[ \varphi_i (k_l) \frac{\mathd k_l}{m^{2} +
     | k_l |^2} \right] \approx 2^{i (- 8 + 9)} \approx 2^i . \]

  Let us now estimate various terms in $R$. The terms containing only
  combinations of $X, X^{\!\resizebox{0.6em}{!}{
\begin{tikzpicture}
\pgfpathmoveto{\pgfqpoint{0cm}{-0.035cm}}
\pgfpathlineto{\pgfqpoint{1.376cm}{-0.035cm}}
\pgfpathlineto{\pgfqpoint{1.376cm}{1.552cm}}
\pgfpathlineto{\pgfqpoint{0cm}{1.552cm}}
\pgfpathclose
\pgfusepath{clip}
\begin{pgfscope}
\begin{pgfscope}
\pgfpathmoveto{\pgfqpoint{0cm}{-0.035cm}}
\pgfpathlineto{\pgfqpoint{1.376cm}{-0.035cm}}
\pgfpathlineto{\pgfqpoint{1.376cm}{1.552cm}}
\pgfpathlineto{\pgfqpoint{0cm}{1.552cm}}
\pgfpathclose
\pgfusepath{clip}
\begin{pgfscope}
\begin{pgfscope}
\pgfsetdash{}{0cm}
\pgfsetlinewidth{0.818mm}
\pgfsetroundcap
\pgfsetroundjoin
\pgfsetmiterlimit{7.0}
\definecolor{eps2pgf_color}{gray}{0}\pgfsetstrokecolor{eps2pgf_color}\pgfsetfillcolor{eps2pgf_color}
\pgfpathmoveto{\pgfqpoint{0.117cm}{1.421cm}}
\pgfpathlineto{\pgfqpoint{0.682cm}{0.671cm}}
\pgfpathlineto{\pgfqpoint{1.246cm}{1.421cm}}
\pgfusepath{stroke}
\end{pgfscope}
\definecolor{eps2pgf_color}{gray}{0}\pgfsetstrokecolor{eps2pgf_color}\pgfsetfillcolor{eps2pgf_color}
\pgfpathmoveto{\pgfqpoint{0.273cm}{1.395cm}}
\pgfpathcurveto{\pgfqpoint{0.273cm}{1.432cm}}{\pgfqpoint{0.259cm}{1.467cm}}{\pgfqpoint{0.233cm}{1.492cm}}
\pgfpathcurveto{\pgfqpoint{0.207cm}{1.518cm}}{\pgfqpoint{0.173cm}{1.532cm}}{\pgfqpoint{0.137cm}{1.532cm}}
\pgfpathcurveto{\pgfqpoint{0.1cm}{1.532cm}}{\pgfqpoint{0.066cm}{1.518cm}}{\pgfqpoint{0.04cm}{1.492cm}}
\pgfpathcurveto{\pgfqpoint{0.014cm}{1.467cm}}{\pgfqpoint{0cm}{1.432cm}}{\pgfqpoint{0cm}{1.395cm}}
\pgfpathcurveto{\pgfqpoint{0cm}{1.359cm}}{\pgfqpoint{0.014cm}{1.324cm}}{\pgfqpoint{0.04cm}{1.299cm}}
\pgfpathcurveto{\pgfqpoint{0.066cm}{1.273cm}}{\pgfqpoint{0.1cm}{1.258cm}}{\pgfqpoint{0.137cm}{1.258cm}}
\pgfpathcurveto{\pgfqpoint{0.173cm}{1.258cm}}{\pgfqpoint{0.207cm}{1.273cm}}{\pgfqpoint{0.233cm}{1.299cm}}
\pgfpathcurveto{\pgfqpoint{0.259cm}{1.324cm}}{\pgfqpoint{0.273cm}{1.359cm}}{\pgfqpoint{0.273cm}{1.395cm}}
\pgfusepath{fill}
\begin{pgfscope}
\pgfsetdash{}{0cm}
\pgfsetlinewidth{0.818mm}
\pgfsetmiterlimit{7.0}
\pgfpathmoveto{\pgfqpoint{0.682cm}{0.671cm}}
\pgfpathlineto{\pgfqpoint{0.679cm}{1.418cm}}
\pgfusepath{stroke}
\end{pgfscope}
\pgfpathmoveto{\pgfqpoint{0.815cm}{1.399cm}}
\pgfpathcurveto{\pgfqpoint{0.815cm}{1.435cm}}{\pgfqpoint{0.801cm}{1.47cm}}{\pgfqpoint{0.775cm}{1.496cm}}
\pgfpathcurveto{\pgfqpoint{0.75cm}{1.521cm}}{\pgfqpoint{0.715cm}{1.536cm}}{\pgfqpoint{0.679cm}{1.536cm}}
\pgfpathcurveto{\pgfqpoint{0.643cm}{1.536cm}}{\pgfqpoint{0.608cm}{1.521cm}}{\pgfqpoint{0.582cm}{1.496cm}}
\pgfpathcurveto{\pgfqpoint{0.557cm}{1.47cm}}{\pgfqpoint{0.542cm}{1.435cm}}{\pgfqpoint{0.542cm}{1.399cm}}
\pgfpathcurveto{\pgfqpoint{0.542cm}{1.363cm}}{\pgfqpoint{0.557cm}{1.328cm}}{\pgfqpoint{0.582cm}{1.302cm}}
\pgfpathcurveto{\pgfqpoint{0.608cm}{1.276cm}}{\pgfqpoint{0.643cm}{1.262cm}}{\pgfqpoint{0.679cm}{1.262cm}}
\pgfpathcurveto{\pgfqpoint{0.715cm}{1.262cm}}{\pgfqpoint{0.75cm}{1.276cm}}{\pgfqpoint{0.775cm}{1.302cm}}
\pgfpathcurveto{\pgfqpoint{0.801cm}{1.328cm}}{\pgfqpoint{0.815cm}{1.363cm}}{\pgfqpoint{0.815cm}{1.399cm}}
\pgfusepath{fill}
\pgfpathmoveto{\pgfqpoint{1.345cm}{1.371cm}}
\pgfpathcurveto{\pgfqpoint{1.345cm}{1.408cm}}{\pgfqpoint{1.331cm}{1.442cm}}{\pgfqpoint{1.305cm}{1.468cm}}
\pgfpathcurveto{\pgfqpoint{1.28cm}{1.494cm}}{\pgfqpoint{1.245cm}{1.508cm}}{\pgfqpoint{1.209cm}{1.508cm}}
\pgfpathcurveto{\pgfqpoint{1.172cm}{1.508cm}}{\pgfqpoint{1.138cm}{1.494cm}}{\pgfqpoint{1.112cm}{1.468cm}}
\pgfpathcurveto{\pgfqpoint{1.087cm}{1.442cm}}{\pgfqpoint{1.072cm}{1.408cm}}{\pgfqpoint{1.072cm}{1.371cm}}
\pgfpathcurveto{\pgfqpoint{1.072cm}{1.335cm}}{\pgfqpoint{1.087cm}{1.3cm}}{\pgfqpoint{1.112cm}{1.274cm}}
\pgfpathcurveto{\pgfqpoint{1.138cm}{1.249cm}}{\pgfqpoint{1.172cm}{1.234cm}}{\pgfqpoint{1.209cm}{1.234cm}}
\pgfpathcurveto{\pgfqpoint{1.245cm}{1.234cm}}{\pgfqpoint{1.28cm}{1.249cm}}{\pgfqpoint{1.305cm}{1.274cm}}
\pgfpathcurveto{\pgfqpoint{1.331cm}{1.3cm}}{\pgfqpoint{1.345cm}{1.335cm}}{\pgfqpoint{1.345cm}{1.371cm}}
\pgfusepath{fill}
\begin{pgfscope}
\pgfsetdash{}{0cm}
\pgfsetlinewidth{0.818mm}
\pgfsetroundcap
\pgfsetmiterlimit{4.0}
\pgfpathmoveto{\pgfqpoint{0.682cm}{0.671cm}}
\pgfpathlineto{\pgfqpoint{0.682cm}{0.042cm}}
\pgfusepath{stroke}
\end{pgfscope}
\end{pgfscope}
\end{pgfscope}
\end{pgfscope}
\end{tikzpicture}}}$ can be estimated directly whereas for
  terms where $\zeta$ appears it is necessary to use stationarity due to the
  limited integrability in space. For instance,
  \[ \left| \mathbb{E} \left[ (\Delta_i X) (0) (\Delta_i X) (0) (
     \Delta_i X^{\!\resizebox{0.6em}{!}{
\begin{tikzpicture}
\pgfpathmoveto{\pgfqpoint{0cm}{-0.035cm}}
\pgfpathlineto{\pgfqpoint{1.376cm}{-0.035cm}}
\pgfpathlineto{\pgfqpoint{1.376cm}{1.552cm}}
\pgfpathlineto{\pgfqpoint{0cm}{1.552cm}}
\pgfpathclose
\pgfusepath{clip}
\begin{pgfscope}
\begin{pgfscope}
\pgfpathmoveto{\pgfqpoint{0cm}{-0.035cm}}
\pgfpathlineto{\pgfqpoint{1.376cm}{-0.035cm}}
\pgfpathlineto{\pgfqpoint{1.376cm}{1.552cm}}
\pgfpathlineto{\pgfqpoint{0cm}{1.552cm}}
\pgfpathclose
\pgfusepath{clip}
\begin{pgfscope}
\begin{pgfscope}
\pgfsetdash{}{0cm}
\pgfsetlinewidth{0.818mm}
\pgfsetroundcap
\pgfsetroundjoin
\pgfsetmiterlimit{7.0}
\definecolor{eps2pgf_color}{gray}{0}\pgfsetstrokecolor{eps2pgf_color}\pgfsetfillcolor{eps2pgf_color}
\pgfpathmoveto{\pgfqpoint{0.117cm}{1.421cm}}
\pgfpathlineto{\pgfqpoint{0.682cm}{0.671cm}}
\pgfpathlineto{\pgfqpoint{1.246cm}{1.421cm}}
\pgfusepath{stroke}
\end{pgfscope}
\definecolor{eps2pgf_color}{gray}{0}\pgfsetstrokecolor{eps2pgf_color}\pgfsetfillcolor{eps2pgf_color}
\pgfpathmoveto{\pgfqpoint{0.273cm}{1.395cm}}
\pgfpathcurveto{\pgfqpoint{0.273cm}{1.432cm}}{\pgfqpoint{0.259cm}{1.467cm}}{\pgfqpoint{0.233cm}{1.492cm}}
\pgfpathcurveto{\pgfqpoint{0.207cm}{1.518cm}}{\pgfqpoint{0.173cm}{1.532cm}}{\pgfqpoint{0.137cm}{1.532cm}}
\pgfpathcurveto{\pgfqpoint{0.1cm}{1.532cm}}{\pgfqpoint{0.066cm}{1.518cm}}{\pgfqpoint{0.04cm}{1.492cm}}
\pgfpathcurveto{\pgfqpoint{0.014cm}{1.467cm}}{\pgfqpoint{0cm}{1.432cm}}{\pgfqpoint{0cm}{1.395cm}}
\pgfpathcurveto{\pgfqpoint{0cm}{1.359cm}}{\pgfqpoint{0.014cm}{1.324cm}}{\pgfqpoint{0.04cm}{1.299cm}}
\pgfpathcurveto{\pgfqpoint{0.066cm}{1.273cm}}{\pgfqpoint{0.1cm}{1.258cm}}{\pgfqpoint{0.137cm}{1.258cm}}
\pgfpathcurveto{\pgfqpoint{0.173cm}{1.258cm}}{\pgfqpoint{0.207cm}{1.273cm}}{\pgfqpoint{0.233cm}{1.299cm}}
\pgfpathcurveto{\pgfqpoint{0.259cm}{1.324cm}}{\pgfqpoint{0.273cm}{1.359cm}}{\pgfqpoint{0.273cm}{1.395cm}}
\pgfusepath{fill}
\begin{pgfscope}
\pgfsetdash{}{0cm}
\pgfsetlinewidth{0.818mm}
\pgfsetmiterlimit{7.0}
\pgfpathmoveto{\pgfqpoint{0.682cm}{0.671cm}}
\pgfpathlineto{\pgfqpoint{0.679cm}{1.418cm}}
\pgfusepath{stroke}
\end{pgfscope}
\pgfpathmoveto{\pgfqpoint{0.815cm}{1.399cm}}
\pgfpathcurveto{\pgfqpoint{0.815cm}{1.435cm}}{\pgfqpoint{0.801cm}{1.47cm}}{\pgfqpoint{0.775cm}{1.496cm}}
\pgfpathcurveto{\pgfqpoint{0.75cm}{1.521cm}}{\pgfqpoint{0.715cm}{1.536cm}}{\pgfqpoint{0.679cm}{1.536cm}}
\pgfpathcurveto{\pgfqpoint{0.643cm}{1.536cm}}{\pgfqpoint{0.608cm}{1.521cm}}{\pgfqpoint{0.582cm}{1.496cm}}
\pgfpathcurveto{\pgfqpoint{0.557cm}{1.47cm}}{\pgfqpoint{0.542cm}{1.435cm}}{\pgfqpoint{0.542cm}{1.399cm}}
\pgfpathcurveto{\pgfqpoint{0.542cm}{1.363cm}}{\pgfqpoint{0.557cm}{1.328cm}}{\pgfqpoint{0.582cm}{1.302cm}}
\pgfpathcurveto{\pgfqpoint{0.608cm}{1.276cm}}{\pgfqpoint{0.643cm}{1.262cm}}{\pgfqpoint{0.679cm}{1.262cm}}
\pgfpathcurveto{\pgfqpoint{0.715cm}{1.262cm}}{\pgfqpoint{0.75cm}{1.276cm}}{\pgfqpoint{0.775cm}{1.302cm}}
\pgfpathcurveto{\pgfqpoint{0.801cm}{1.328cm}}{\pgfqpoint{0.815cm}{1.363cm}}{\pgfqpoint{0.815cm}{1.399cm}}
\pgfusepath{fill}
\pgfpathmoveto{\pgfqpoint{1.345cm}{1.371cm}}
\pgfpathcurveto{\pgfqpoint{1.345cm}{1.408cm}}{\pgfqpoint{1.331cm}{1.442cm}}{\pgfqpoint{1.305cm}{1.468cm}}
\pgfpathcurveto{\pgfqpoint{1.28cm}{1.494cm}}{\pgfqpoint{1.245cm}{1.508cm}}{\pgfqpoint{1.209cm}{1.508cm}}
\pgfpathcurveto{\pgfqpoint{1.172cm}{1.508cm}}{\pgfqpoint{1.138cm}{1.494cm}}{\pgfqpoint{1.112cm}{1.468cm}}
\pgfpathcurveto{\pgfqpoint{1.087cm}{1.442cm}}{\pgfqpoint{1.072cm}{1.408cm}}{\pgfqpoint{1.072cm}{1.371cm}}
\pgfpathcurveto{\pgfqpoint{1.072cm}{1.335cm}}{\pgfqpoint{1.087cm}{1.3cm}}{\pgfqpoint{1.112cm}{1.274cm}}
\pgfpathcurveto{\pgfqpoint{1.138cm}{1.249cm}}{\pgfqpoint{1.172cm}{1.234cm}}{\pgfqpoint{1.209cm}{1.234cm}}
\pgfpathcurveto{\pgfqpoint{1.245cm}{1.234cm}}{\pgfqpoint{1.28cm}{1.249cm}}{\pgfqpoint{1.305cm}{1.274cm}}
\pgfpathcurveto{\pgfqpoint{1.331cm}{1.3cm}}{\pgfqpoint{1.345cm}{1.335cm}}{\pgfqpoint{1.345cm}{1.371cm}}
\pgfusepath{fill}
\begin{pgfscope}
\pgfsetdash{}{0cm}
\pgfsetlinewidth{0.818mm}
\pgfsetroundcap
\pgfsetmiterlimit{4.0}
\pgfpathmoveto{\pgfqpoint{0.682cm}{0.671cm}}
\pgfpathlineto{\pgfqpoint{0.682cm}{0.042cm}}
\pgfusepath{stroke}
\end{pgfscope}
\end{pgfscope}
\end{pgfscope}
\end{pgfscope}
\end{tikzpicture}}} ) (0) ( \Delta_i X^{\!\resizebox{0.6em}{!}{
\begin{tikzpicture}
\pgfpathmoveto{\pgfqpoint{0cm}{-0.035cm}}
\pgfpathlineto{\pgfqpoint{1.376cm}{-0.035cm}}
\pgfpathlineto{\pgfqpoint{1.376cm}{1.552cm}}
\pgfpathlineto{\pgfqpoint{0cm}{1.552cm}}
\pgfpathclose
\pgfusepath{clip}
\begin{pgfscope}
\begin{pgfscope}
\pgfpathmoveto{\pgfqpoint{0cm}{-0.035cm}}
\pgfpathlineto{\pgfqpoint{1.376cm}{-0.035cm}}
\pgfpathlineto{\pgfqpoint{1.376cm}{1.552cm}}
\pgfpathlineto{\pgfqpoint{0cm}{1.552cm}}
\pgfpathclose
\pgfusepath{clip}
\begin{pgfscope}
\begin{pgfscope}
\pgfsetdash{}{0cm}
\pgfsetlinewidth{0.818mm}
\pgfsetroundcap
\pgfsetroundjoin
\pgfsetmiterlimit{7.0}
\definecolor{eps2pgf_color}{gray}{0}\pgfsetstrokecolor{eps2pgf_color}\pgfsetfillcolor{eps2pgf_color}
\pgfpathmoveto{\pgfqpoint{0.117cm}{1.421cm}}
\pgfpathlineto{\pgfqpoint{0.682cm}{0.671cm}}
\pgfpathlineto{\pgfqpoint{1.246cm}{1.421cm}}
\pgfusepath{stroke}
\end{pgfscope}
\definecolor{eps2pgf_color}{gray}{0}\pgfsetstrokecolor{eps2pgf_color}\pgfsetfillcolor{eps2pgf_color}
\pgfpathmoveto{\pgfqpoint{0.273cm}{1.395cm}}
\pgfpathcurveto{\pgfqpoint{0.273cm}{1.432cm}}{\pgfqpoint{0.259cm}{1.467cm}}{\pgfqpoint{0.233cm}{1.492cm}}
\pgfpathcurveto{\pgfqpoint{0.207cm}{1.518cm}}{\pgfqpoint{0.173cm}{1.532cm}}{\pgfqpoint{0.137cm}{1.532cm}}
\pgfpathcurveto{\pgfqpoint{0.1cm}{1.532cm}}{\pgfqpoint{0.066cm}{1.518cm}}{\pgfqpoint{0.04cm}{1.492cm}}
\pgfpathcurveto{\pgfqpoint{0.014cm}{1.467cm}}{\pgfqpoint{0cm}{1.432cm}}{\pgfqpoint{0cm}{1.395cm}}
\pgfpathcurveto{\pgfqpoint{0cm}{1.359cm}}{\pgfqpoint{0.014cm}{1.324cm}}{\pgfqpoint{0.04cm}{1.299cm}}
\pgfpathcurveto{\pgfqpoint{0.066cm}{1.273cm}}{\pgfqpoint{0.1cm}{1.258cm}}{\pgfqpoint{0.137cm}{1.258cm}}
\pgfpathcurveto{\pgfqpoint{0.173cm}{1.258cm}}{\pgfqpoint{0.207cm}{1.273cm}}{\pgfqpoint{0.233cm}{1.299cm}}
\pgfpathcurveto{\pgfqpoint{0.259cm}{1.324cm}}{\pgfqpoint{0.273cm}{1.359cm}}{\pgfqpoint{0.273cm}{1.395cm}}
\pgfusepath{fill}
\begin{pgfscope}
\pgfsetdash{}{0cm}
\pgfsetlinewidth{0.818mm}
\pgfsetmiterlimit{7.0}
\pgfpathmoveto{\pgfqpoint{0.682cm}{0.671cm}}
\pgfpathlineto{\pgfqpoint{0.679cm}{1.418cm}}
\pgfusepath{stroke}
\end{pgfscope}
\pgfpathmoveto{\pgfqpoint{0.815cm}{1.399cm}}
\pgfpathcurveto{\pgfqpoint{0.815cm}{1.435cm}}{\pgfqpoint{0.801cm}{1.47cm}}{\pgfqpoint{0.775cm}{1.496cm}}
\pgfpathcurveto{\pgfqpoint{0.75cm}{1.521cm}}{\pgfqpoint{0.715cm}{1.536cm}}{\pgfqpoint{0.679cm}{1.536cm}}
\pgfpathcurveto{\pgfqpoint{0.643cm}{1.536cm}}{\pgfqpoint{0.608cm}{1.521cm}}{\pgfqpoint{0.582cm}{1.496cm}}
\pgfpathcurveto{\pgfqpoint{0.557cm}{1.47cm}}{\pgfqpoint{0.542cm}{1.435cm}}{\pgfqpoint{0.542cm}{1.399cm}}
\pgfpathcurveto{\pgfqpoint{0.542cm}{1.363cm}}{\pgfqpoint{0.557cm}{1.328cm}}{\pgfqpoint{0.582cm}{1.302cm}}
\pgfpathcurveto{\pgfqpoint{0.608cm}{1.276cm}}{\pgfqpoint{0.643cm}{1.262cm}}{\pgfqpoint{0.679cm}{1.262cm}}
\pgfpathcurveto{\pgfqpoint{0.715cm}{1.262cm}}{\pgfqpoint{0.75cm}{1.276cm}}{\pgfqpoint{0.775cm}{1.302cm}}
\pgfpathcurveto{\pgfqpoint{0.801cm}{1.328cm}}{\pgfqpoint{0.815cm}{1.363cm}}{\pgfqpoint{0.815cm}{1.399cm}}
\pgfusepath{fill}
\pgfpathmoveto{\pgfqpoint{1.345cm}{1.371cm}}
\pgfpathcurveto{\pgfqpoint{1.345cm}{1.408cm}}{\pgfqpoint{1.331cm}{1.442cm}}{\pgfqpoint{1.305cm}{1.468cm}}
\pgfpathcurveto{\pgfqpoint{1.28cm}{1.494cm}}{\pgfqpoint{1.245cm}{1.508cm}}{\pgfqpoint{1.209cm}{1.508cm}}
\pgfpathcurveto{\pgfqpoint{1.172cm}{1.508cm}}{\pgfqpoint{1.138cm}{1.494cm}}{\pgfqpoint{1.112cm}{1.468cm}}
\pgfpathcurveto{\pgfqpoint{1.087cm}{1.442cm}}{\pgfqpoint{1.072cm}{1.408cm}}{\pgfqpoint{1.072cm}{1.371cm}}
\pgfpathcurveto{\pgfqpoint{1.072cm}{1.335cm}}{\pgfqpoint{1.087cm}{1.3cm}}{\pgfqpoint{1.112cm}{1.274cm}}
\pgfpathcurveto{\pgfqpoint{1.138cm}{1.249cm}}{\pgfqpoint{1.172cm}{1.234cm}}{\pgfqpoint{1.209cm}{1.234cm}}
\pgfpathcurveto{\pgfqpoint{1.245cm}{1.234cm}}{\pgfqpoint{1.28cm}{1.249cm}}{\pgfqpoint{1.305cm}{1.274cm}}
\pgfpathcurveto{\pgfqpoint{1.331cm}{1.3cm}}{\pgfqpoint{1.345cm}{1.335cm}}{\pgfqpoint{1.345cm}{1.371cm}}
\pgfusepath{fill}
\begin{pgfscope}
\pgfsetdash{}{0cm}
\pgfsetlinewidth{0.818mm}
\pgfsetroundcap
\pgfsetmiterlimit{4.0}
\pgfpathmoveto{\pgfqpoint{0.682cm}{0.671cm}}
\pgfpathlineto{\pgfqpoint{0.682cm}{0.042cm}}
\pgfusepath{stroke}
\end{pgfscope}
\end{pgfscope}
\end{pgfscope}
\end{pgfscope}
\end{tikzpicture}}} )
     (0) \right] \right| \]
  \[ \lesssim 2^{- 2 i (- 1 / 2 - \kappa)} 2^{- 2 i (1 / 2 - \kappa)}
     \mathbb{E} \left[ \| X \|_{\CC^{- 1 / 2 - \kappa} (\rho^{\sigma})}^2
     \| X^{\!\resizebox{0.6em}{!}{
\begin{tikzpicture}
\pgfpathmoveto{\pgfqpoint{0cm}{-0.035cm}}
\pgfpathlineto{\pgfqpoint{1.376cm}{-0.035cm}}
\pgfpathlineto{\pgfqpoint{1.376cm}{1.552cm}}
\pgfpathlineto{\pgfqpoint{0cm}{1.552cm}}
\pgfpathclose
\pgfusepath{clip}
\begin{pgfscope}
\begin{pgfscope}
\pgfpathmoveto{\pgfqpoint{0cm}{-0.035cm}}
\pgfpathlineto{\pgfqpoint{1.376cm}{-0.035cm}}
\pgfpathlineto{\pgfqpoint{1.376cm}{1.552cm}}
\pgfpathlineto{\pgfqpoint{0cm}{1.552cm}}
\pgfpathclose
\pgfusepath{clip}
\begin{pgfscope}
\begin{pgfscope}
\pgfsetdash{}{0cm}
\pgfsetlinewidth{0.818mm}
\pgfsetroundcap
\pgfsetroundjoin
\pgfsetmiterlimit{7.0}
\definecolor{eps2pgf_color}{gray}{0}\pgfsetstrokecolor{eps2pgf_color}\pgfsetfillcolor{eps2pgf_color}
\pgfpathmoveto{\pgfqpoint{0.117cm}{1.421cm}}
\pgfpathlineto{\pgfqpoint{0.682cm}{0.671cm}}
\pgfpathlineto{\pgfqpoint{1.246cm}{1.421cm}}
\pgfusepath{stroke}
\end{pgfscope}
\definecolor{eps2pgf_color}{gray}{0}\pgfsetstrokecolor{eps2pgf_color}\pgfsetfillcolor{eps2pgf_color}
\pgfpathmoveto{\pgfqpoint{0.273cm}{1.395cm}}
\pgfpathcurveto{\pgfqpoint{0.273cm}{1.432cm}}{\pgfqpoint{0.259cm}{1.467cm}}{\pgfqpoint{0.233cm}{1.492cm}}
\pgfpathcurveto{\pgfqpoint{0.207cm}{1.518cm}}{\pgfqpoint{0.173cm}{1.532cm}}{\pgfqpoint{0.137cm}{1.532cm}}
\pgfpathcurveto{\pgfqpoint{0.1cm}{1.532cm}}{\pgfqpoint{0.066cm}{1.518cm}}{\pgfqpoint{0.04cm}{1.492cm}}
\pgfpathcurveto{\pgfqpoint{0.014cm}{1.467cm}}{\pgfqpoint{0cm}{1.432cm}}{\pgfqpoint{0cm}{1.395cm}}
\pgfpathcurveto{\pgfqpoint{0cm}{1.359cm}}{\pgfqpoint{0.014cm}{1.324cm}}{\pgfqpoint{0.04cm}{1.299cm}}
\pgfpathcurveto{\pgfqpoint{0.066cm}{1.273cm}}{\pgfqpoint{0.1cm}{1.258cm}}{\pgfqpoint{0.137cm}{1.258cm}}
\pgfpathcurveto{\pgfqpoint{0.173cm}{1.258cm}}{\pgfqpoint{0.207cm}{1.273cm}}{\pgfqpoint{0.233cm}{1.299cm}}
\pgfpathcurveto{\pgfqpoint{0.259cm}{1.324cm}}{\pgfqpoint{0.273cm}{1.359cm}}{\pgfqpoint{0.273cm}{1.395cm}}
\pgfusepath{fill}
\begin{pgfscope}
\pgfsetdash{}{0cm}
\pgfsetlinewidth{0.818mm}
\pgfsetmiterlimit{7.0}
\pgfpathmoveto{\pgfqpoint{0.682cm}{0.671cm}}
\pgfpathlineto{\pgfqpoint{0.679cm}{1.418cm}}
\pgfusepath{stroke}
\end{pgfscope}
\pgfpathmoveto{\pgfqpoint{0.815cm}{1.399cm}}
\pgfpathcurveto{\pgfqpoint{0.815cm}{1.435cm}}{\pgfqpoint{0.801cm}{1.47cm}}{\pgfqpoint{0.775cm}{1.496cm}}
\pgfpathcurveto{\pgfqpoint{0.75cm}{1.521cm}}{\pgfqpoint{0.715cm}{1.536cm}}{\pgfqpoint{0.679cm}{1.536cm}}
\pgfpathcurveto{\pgfqpoint{0.643cm}{1.536cm}}{\pgfqpoint{0.608cm}{1.521cm}}{\pgfqpoint{0.582cm}{1.496cm}}
\pgfpathcurveto{\pgfqpoint{0.557cm}{1.47cm}}{\pgfqpoint{0.542cm}{1.435cm}}{\pgfqpoint{0.542cm}{1.399cm}}
\pgfpathcurveto{\pgfqpoint{0.542cm}{1.363cm}}{\pgfqpoint{0.557cm}{1.328cm}}{\pgfqpoint{0.582cm}{1.302cm}}
\pgfpathcurveto{\pgfqpoint{0.608cm}{1.276cm}}{\pgfqpoint{0.643cm}{1.262cm}}{\pgfqpoint{0.679cm}{1.262cm}}
\pgfpathcurveto{\pgfqpoint{0.715cm}{1.262cm}}{\pgfqpoint{0.75cm}{1.276cm}}{\pgfqpoint{0.775cm}{1.302cm}}
\pgfpathcurveto{\pgfqpoint{0.801cm}{1.328cm}}{\pgfqpoint{0.815cm}{1.363cm}}{\pgfqpoint{0.815cm}{1.399cm}}
\pgfusepath{fill}
\pgfpathmoveto{\pgfqpoint{1.345cm}{1.371cm}}
\pgfpathcurveto{\pgfqpoint{1.345cm}{1.408cm}}{\pgfqpoint{1.331cm}{1.442cm}}{\pgfqpoint{1.305cm}{1.468cm}}
\pgfpathcurveto{\pgfqpoint{1.28cm}{1.494cm}}{\pgfqpoint{1.245cm}{1.508cm}}{\pgfqpoint{1.209cm}{1.508cm}}
\pgfpathcurveto{\pgfqpoint{1.172cm}{1.508cm}}{\pgfqpoint{1.138cm}{1.494cm}}{\pgfqpoint{1.112cm}{1.468cm}}
\pgfpathcurveto{\pgfqpoint{1.087cm}{1.442cm}}{\pgfqpoint{1.072cm}{1.408cm}}{\pgfqpoint{1.072cm}{1.371cm}}
\pgfpathcurveto{\pgfqpoint{1.072cm}{1.335cm}}{\pgfqpoint{1.087cm}{1.3cm}}{\pgfqpoint{1.112cm}{1.274cm}}
\pgfpathcurveto{\pgfqpoint{1.138cm}{1.249cm}}{\pgfqpoint{1.172cm}{1.234cm}}{\pgfqpoint{1.209cm}{1.234cm}}
\pgfpathcurveto{\pgfqpoint{1.245cm}{1.234cm}}{\pgfqpoint{1.28cm}{1.249cm}}{\pgfqpoint{1.305cm}{1.274cm}}
\pgfpathcurveto{\pgfqpoint{1.331cm}{1.3cm}}{\pgfqpoint{1.345cm}{1.335cm}}{\pgfqpoint{1.345cm}{1.371cm}}
\pgfusepath{fill}
\begin{pgfscope}
\pgfsetdash{}{0cm}
\pgfsetlinewidth{0.818mm}
\pgfsetroundcap
\pgfsetmiterlimit{4.0}
\pgfpathmoveto{\pgfqpoint{0.682cm}{0.671cm}}
\pgfpathlineto{\pgfqpoint{0.682cm}{0.042cm}}
\pgfusepath{stroke}
\end{pgfscope}
\end{pgfscope}
\end{pgfscope}
\end{pgfscope}
\end{tikzpicture}}} \|_{\CC^{1 / 2 - \kappa} (\rho^{\sigma})}^2
     \right] \lesssim 2^{i4 \kappa} \]
  and similarly for the other terms without $\zeta$ which are collectively of order $2^{i4 \kappa} (\lambda^2+\lambda^4)$. For the remaining terms,
  we fix a weight $\rho$ as above and use stationarity. In addition, we shall
  be careful about having the necessary integrability. For instance, for the
  most irregular term we have
  \[ \mathbb{E} [(\Delta_i X)^3 (0) (\Delta_i \zeta) (0)] =
     \int_{\mathbb{R}^d} \rho^4 (x) \mathbb{E} [(\Delta_i X)^3 (x) (\Delta_i
     \zeta) (x)] \mathd x =\mathbb{E} \langle \rho^4, (\Delta_i X)^3 (\Delta_i
     \zeta) \rangle \]
  and we bound this quantity as
  \[\begin{aligned}
   | \mathbb{E} [(\Delta_i X)^3 (0) (\Delta_i \zeta) (0)] | & \leqslant
     \mathbb{E} [\| \Delta_i X_{\varepsilon} \|_{L^{\infty} (\rho^{\sigma})}^3
     \| \Delta_i \zeta \|_{L^1 (\rho^{4 - 3 \sigma})}] \lesssim \mathbb{E} [\|
     \Delta_i X_{\varepsilon} \|_{L^{\infty} (\rho^{\sigma})}^3 \| \Delta_i
     \zeta \|_{L^2 (\rho^2)}] \\
  & \lesssim 2^{- 3 i (- 1 / 2 - \kappa)} 2^{i (- 1 + 2 \kappa)} \mathbb{E}
     \left[ \| X \|^3_{\CC^{- 1 / 2 - \kappa} (\rho^{\sigma})} \| \zeta
     \|_{B^{1 - 2 \kappa}_{2, 2} (\rho^2)} \right]  \\
  & \lesssim 2^{- 3 i (- 1 / 2 - \kappa)} 2^{i (- 1 + 2 \kappa)} (\mathbb{E}
     [ \| X \|^6_{\CC^{- 1 / 2 - \kappa} (\rho^{\sigma})}])^{1/2} (\mathbb{E}[\| \zeta
     \|^2_{B^{1 - 2 \kappa}_{2, 2} (\rho^2)}])^{1/2} 
     \\ & \lesssim 2^{i (1 / 2 + 5
     \kappa)} (\lambda+\lambda^{7/2}).
     \end{aligned} 
     \]
 where we used Theorem~\ref{thm:main}. Next,
  \[ | \mathbb{E} [(\Delta_i X)^2 (0) (\Delta_i \zeta)^2 (0)] | \leqslant
     \mathbb{E} [\| \Delta_i X \|^2_{L^{\infty} (\rho^{\sigma})} \| \Delta_i
     \zeta \|_{L^2 (\rho^{1 + \iota})} \| \Delta_i \zeta \|_{L^2 (\rho^2)}] \]
  \[ \leqslant 2^{- 2 i (- 1 / 2 - \kappa)} 2^{- i (1 - 2 \kappa)} \mathbb{E}
     [\| X \|^2_{\CC^{-1/2-\kappa} (\rho^{\sigma})} \| \zeta \|_{B^0_{4, \infty}
     (\rho)} \|  \zeta \|_{H^{1 - 2 \kappa} (\rho^2)}] \lesssim 2^{i 4
     \kappa} ({\lambda^{5/4}+\lambda^{5}}), \]
  and
  \[ 
  \begin{aligned}
  | \mathbb{E} [(\Delta_i X) (0) (\Delta_i \zeta)^3 (0)] | & \leqslant
     \mathbb{E} [\| \Delta_i X \|_{L^{\infty} (\rho^{\sigma})} \| \Delta_i
     \zeta \|^3_{L^3 (\rho^{(4 - \sigma) / 3})}] \\
  & \leqslant \mathbb{E} [\| \Delta_i X \|_{L^{\infty} (\rho^{\sigma})} \|
     \Delta_i \zeta \|^3_{L^4 (\rho)}] \\ & \quad \lesssim 2^{- i (- 1 / 2 - \kappa)}
     \mathbb{E} \left[ \| X \|_{\CC^{- 1 / 2 - \kappa} (\rho^{\sigma})} \|
     \zeta \|^3_{B^0_{4, \infty} (\rho)} \right] \\
   & \lesssim 2^{i (1 / 2 +
     \kappa)}({\lambda^{3/4}+\lambda^{9/2}}),
     \end{aligned} \]
  \[ | \mathbb{E} [(\Delta_i \zeta)^4 (0)] | = | \mathbb{E} \langle \rho^4,
     (\Delta_i \zeta)^4 \rangle | \leqslant \mathbb{E} \| (\Delta_i \zeta)
     \|^4_{L^4 (\rho)} \leqslant \mathbb{E} [\| \zeta \|^4_{B^0_{4, \infty}
     (\rho)}] \lesssim ({\lambda+\lambda^{6}}). \]
  Proceeding similarly for the other terms we finally obtain the  bound   \[ | R | \lesssim  2^{i (1 / 2 + 5 \kappa)} ({\lambda^{3/4}+\lambda^{7}}). \]
  Therefore for a fixed $\lambda>0$ there exists  a sufficiently large $i$  such that
  \[ \mathbb{E} [(\Delta_i \varphi)^4 (0)] - 3 (\mathbb{E} [(\Delta_i
     \varphi)^2 (0)^2])^2 \lesssim -2^i \lambda < 0, \]
  and the proof is complete.
\end{proof}

\begin{remark}\label{rem:rev1}
To our knowledge, the proof of non-Gaussianity given above, is new. In particular the pathwise estimates of the PDE methods allow to probe correlation functions at high-momenta and check that they are, at leading order, given by perturbative contributions irrespective of the size of the coupling $\lambda$. This seems to be a substantial improvement with respect to the perturbative strategy of~\cite{MR723546} which requires small $\lambda$.
\end{remark}

\section{Integration by parts formula and Dyson--Schwinger equations}
\label{s:sd}

The goal of this section is twofold. First, we introduce a new
paracontrolled ansatz, which allows to prove higher regularity and in
particular to give meaning to the critical resonant product in the continuum.
Second, the higher regularity is used in order to improve the tightness 
and  to construct a renormalized cubic term $\llbracket
\varphi^3 \rrbracket$. 
Finally, we derive an integration by parts formula
together with the Dyson--Schwinger equations and we identify the continuum dynamics.

\subsection{Improved tightness}
\label{s:reg}

In this section we establish  higher order regularity and a better tightness which is needed in order
to define the resonant product $\llbracket X^2 \rrbracket \circ \phi$ in the
continuum limit. Recall that the equation {\eqref{eq:phi-eq}} satisfied by
$\phi_{M, \varepsilon}$ has the form
\begin{equation}
  \LL_{\varepsilon} \phi_{M, \varepsilon} = - 3 \lambda \llbracket X_{M,
  \varepsilon}^2 \rrbracket \succ \phi_{M, \varepsilon} + U_{M, \varepsilon},
  \label{eq:phiU}
\end{equation}
where
\[ \begin{array}{lll}
     U_{M, \varepsilon} & \assign & - 3 \lambda \llbracket X_{M, \varepsilon}^2
     \rrbracket \preccurlyeq (Y_{M, \varepsilon} + \phi_{M, \varepsilon}) - 3
     \lambda^2 b_{M, \varepsilon} (X_{M, \varepsilon} + Y_{M, \varepsilon} + \phi_{M,
     \varepsilon})\\
     &  & - 3 \lambda ( \UU^{\varepsilon}_{\leqslant} \llbracket X_{M,
     \varepsilon}^2 \rrbracket ) \succ Y_{M, \varepsilon} - 3\lambda X_{M,
     \varepsilon} (Y_{M, \varepsilon} + \phi_{M, \varepsilon})^2 -\lambda Y_{M,
     \varepsilon}^3 \\
     & & - 3\lambda Y_{M, \varepsilon}^2 \phi_{M, \varepsilon} - 3\lambda Y_{M,
     \varepsilon} \phi_{M, \varepsilon}^2 -\lambda \phi_{M, \varepsilon}^3 .
   \end{array} \]
If we let
\begin{equation}
  \chi_{M, \varepsilon} \assign \phi_{M, \varepsilon} + 3\lambda X_{M,
  \varepsilon}^{\!\resizebox{0.6em}{!}{
\begin{tikzpicture}
\pgfpathmoveto{\pgfqpoint{0cm}{0cm}}
\pgfpathlineto{\pgfqpoint{1.376cm}{0cm}}
\pgfpathlineto{\pgfqpoint{1.376cm}{1.588cm}}
\pgfpathlineto{\pgfqpoint{0cm}{1.588cm}}
\pgfpathclose
\pgfusepath{clip}
\begin{pgfscope}
\begin{pgfscope}
\pgfpathmoveto{\pgfqpoint{0cm}{0cm}}
\pgfpathlineto{\pgfqpoint{1.376cm}{0cm}}
\pgfpathlineto{\pgfqpoint{1.376cm}{1.588cm}}
\pgfpathlineto{\pgfqpoint{0cm}{1.588cm}}
\pgfpathclose
\pgfusepath{clip}
\begin{pgfscope}
\begin{pgfscope}
\definecolor{eps2pgf_color}{gray}{0.976471}\pgfsetstrokecolor{eps2pgf_color}\pgfsetfillcolor{eps2pgf_color}
\pgfpathmoveto{\pgfqpoint{0cm}{0cm}}
\pgfpathlineto{\pgfqpoint{1.376cm}{0cm}}
\pgfpathlineto{\pgfqpoint{1.376cm}{1.588cm}}
\pgfpathlineto{\pgfqpoint{0cm}{1.588cm}}
\pgfpathclose
\pgfusepath{fill}
\end{pgfscope}
\begin{pgfscope}
\pgfsetdash{}{0cm}
\pgfsetlinewidth{0.818mm}
\pgfsetroundcap
\pgfsetroundjoin
\pgfsetmiterlimit{7.0}
\definecolor{eps2pgf_color}{gray}{0}\pgfsetstrokecolor{eps2pgf_color}\pgfsetfillcolor{eps2pgf_color}
\pgfpathmoveto{\pgfqpoint{0.117cm}{1.476cm}}
\pgfpathlineto{\pgfqpoint{0.682cm}{0.726cm}}
\pgfpathlineto{\pgfqpoint{1.246cm}{1.476cm}}
\pgfusepath{stroke}
\end{pgfscope}
\definecolor{eps2pgf_color}{gray}{0}\pgfsetstrokecolor{eps2pgf_color}\pgfsetfillcolor{eps2pgf_color}
\pgfpathmoveto{\pgfqpoint{0.273cm}{1.451cm}}
\pgfpathcurveto{\pgfqpoint{0.273cm}{1.487cm}}{\pgfqpoint{0.259cm}{1.522cm}}{\pgfqpoint{0.233cm}{1.547cm}}
\pgfpathcurveto{\pgfqpoint{0.207cm}{1.573cm}}{\pgfqpoint{0.173cm}{1.588cm}}{\pgfqpoint{0.137cm}{1.588cm}}
\pgfpathcurveto{\pgfqpoint{0.1cm}{1.588cm}}{\pgfqpoint{0.066cm}{1.573cm}}{\pgfqpoint{0.04cm}{1.547cm}}
\pgfpathcurveto{\pgfqpoint{0.014cm}{1.522cm}}{\pgfqpoint{0cm}{1.487cm}}{\pgfqpoint{0cm}{1.451cm}}
\pgfpathcurveto{\pgfqpoint{0cm}{1.414cm}}{\pgfqpoint{0.014cm}{1.379cm}}{\pgfqpoint{0.04cm}{1.354cm}}
\pgfpathcurveto{\pgfqpoint{0.066cm}{1.328cm}}{\pgfqpoint{0.1cm}{1.314cm}}{\pgfqpoint{0.137cm}{1.314cm}}
\pgfpathcurveto{\pgfqpoint{0.173cm}{1.314cm}}{\pgfqpoint{0.207cm}{1.328cm}}{\pgfqpoint{0.233cm}{1.354cm}}
\pgfpathcurveto{\pgfqpoint{0.259cm}{1.379cm}}{\pgfqpoint{0.273cm}{1.414cm}}{\pgfqpoint{0.273cm}{1.451cm}}
\pgfusepath{fill}
\pgfpathmoveto{\pgfqpoint{1.345cm}{1.426cm}}
\pgfpathcurveto{\pgfqpoint{1.345cm}{1.463cm}}{\pgfqpoint{1.331cm}{1.497cm}}{\pgfqpoint{1.305cm}{1.523cm}}
\pgfpathcurveto{\pgfqpoint{1.28cm}{1.549cm}}{\pgfqpoint{1.245cm}{1.563cm}}{\pgfqpoint{1.209cm}{1.563cm}}
\pgfpathcurveto{\pgfqpoint{1.172cm}{1.563cm}}{\pgfqpoint{1.138cm}{1.549cm}}{\pgfqpoint{1.112cm}{1.523cm}}
\pgfpathcurveto{\pgfqpoint{1.087cm}{1.497cm}}{\pgfqpoint{1.072cm}{1.463cm}}{\pgfqpoint{1.072cm}{1.426cm}}
\pgfpathcurveto{\pgfqpoint{1.072cm}{1.39cm}}{\pgfqpoint{1.087cm}{1.355cm}}{\pgfqpoint{1.112cm}{1.329cm}}
\pgfpathcurveto{\pgfqpoint{1.138cm}{1.304cm}}{\pgfqpoint{1.172cm}{1.289cm}}{\pgfqpoint{1.209cm}{1.289cm}}
\pgfpathcurveto{\pgfqpoint{1.245cm}{1.289cm}}{\pgfqpoint{1.28cm}{1.304cm}}{\pgfqpoint{1.305cm}{1.329cm}}
\pgfpathcurveto{\pgfqpoint{1.331cm}{1.355cm}}{\pgfqpoint{1.345cm}{1.39cm}}{\pgfqpoint{1.345cm}{1.426cm}}
\pgfusepath{fill}
\begin{pgfscope}
\pgfsetdash{}{0cm}
\pgfsetlinewidth{0.818mm}
\pgfsetroundcap
\pgfsetmiterlimit{4.0}
\pgfpathmoveto{\pgfqpoint{0.682cm}{0.726cm}}
\pgfpathlineto{\pgfqpoint{0.682cm}{0.097cm}}
\pgfusepath{stroke}
\end{pgfscope}
\end{pgfscope}
\end{pgfscope}
\end{pgfscope}
\end{tikzpicture}}} \succ \phi_{M, \varepsilon} \label{eq:chi1},
\end{equation}
we obtain by the commutator lemma, Lemma~\ref{lem:comm1},
\begin{equation*}
  \begin{aligned}
    3\lambda \llbracket X_{M, \varepsilon}^2 \rrbracket \circ \phi_{M, \varepsilon} +
    3\lambda^2 b_{M, \varepsilon} \phi_{M, \varepsilon} & =  - 3\lambda \llbracket X_{M,
    \varepsilon}^2 \rrbracket \circ (3\lambda X_{M, \varepsilon}^{\!\resizebox{0.6em}{!}{
\begin{tikzpicture}
\pgfpathmoveto{\pgfqpoint{0cm}{0cm}}
\pgfpathlineto{\pgfqpoint{1.376cm}{0cm}}
\pgfpathlineto{\pgfqpoint{1.376cm}{1.588cm}}
\pgfpathlineto{\pgfqpoint{0cm}{1.588cm}}
\pgfpathclose
\pgfusepath{clip}
\begin{pgfscope}
\begin{pgfscope}
\pgfpathmoveto{\pgfqpoint{0cm}{0cm}}
\pgfpathlineto{\pgfqpoint{1.376cm}{0cm}}
\pgfpathlineto{\pgfqpoint{1.376cm}{1.588cm}}
\pgfpathlineto{\pgfqpoint{0cm}{1.588cm}}
\pgfpathclose
\pgfusepath{clip}
\begin{pgfscope}
\begin{pgfscope}
\definecolor{eps2pgf_color}{gray}{0.976471}\pgfsetstrokecolor{eps2pgf_color}\pgfsetfillcolor{eps2pgf_color}
\pgfpathmoveto{\pgfqpoint{0cm}{0cm}}
\pgfpathlineto{\pgfqpoint{1.376cm}{0cm}}
\pgfpathlineto{\pgfqpoint{1.376cm}{1.588cm}}
\pgfpathlineto{\pgfqpoint{0cm}{1.588cm}}
\pgfpathclose
\pgfusepath{fill}
\end{pgfscope}
\begin{pgfscope}
\pgfsetdash{}{0cm}
\pgfsetlinewidth{0.818mm}
\pgfsetroundcap
\pgfsetroundjoin
\pgfsetmiterlimit{7.0}
\definecolor{eps2pgf_color}{gray}{0}\pgfsetstrokecolor{eps2pgf_color}\pgfsetfillcolor{eps2pgf_color}
\pgfpathmoveto{\pgfqpoint{0.117cm}{1.476cm}}
\pgfpathlineto{\pgfqpoint{0.682cm}{0.726cm}}
\pgfpathlineto{\pgfqpoint{1.246cm}{1.476cm}}
\pgfusepath{stroke}
\end{pgfscope}
\definecolor{eps2pgf_color}{gray}{0}\pgfsetstrokecolor{eps2pgf_color}\pgfsetfillcolor{eps2pgf_color}
\pgfpathmoveto{\pgfqpoint{0.273cm}{1.451cm}}
\pgfpathcurveto{\pgfqpoint{0.273cm}{1.487cm}}{\pgfqpoint{0.259cm}{1.522cm}}{\pgfqpoint{0.233cm}{1.547cm}}
\pgfpathcurveto{\pgfqpoint{0.207cm}{1.573cm}}{\pgfqpoint{0.173cm}{1.588cm}}{\pgfqpoint{0.137cm}{1.588cm}}
\pgfpathcurveto{\pgfqpoint{0.1cm}{1.588cm}}{\pgfqpoint{0.066cm}{1.573cm}}{\pgfqpoint{0.04cm}{1.547cm}}
\pgfpathcurveto{\pgfqpoint{0.014cm}{1.522cm}}{\pgfqpoint{0cm}{1.487cm}}{\pgfqpoint{0cm}{1.451cm}}
\pgfpathcurveto{\pgfqpoint{0cm}{1.414cm}}{\pgfqpoint{0.014cm}{1.379cm}}{\pgfqpoint{0.04cm}{1.354cm}}
\pgfpathcurveto{\pgfqpoint{0.066cm}{1.328cm}}{\pgfqpoint{0.1cm}{1.314cm}}{\pgfqpoint{0.137cm}{1.314cm}}
\pgfpathcurveto{\pgfqpoint{0.173cm}{1.314cm}}{\pgfqpoint{0.207cm}{1.328cm}}{\pgfqpoint{0.233cm}{1.354cm}}
\pgfpathcurveto{\pgfqpoint{0.259cm}{1.379cm}}{\pgfqpoint{0.273cm}{1.414cm}}{\pgfqpoint{0.273cm}{1.451cm}}
\pgfusepath{fill}
\pgfpathmoveto{\pgfqpoint{1.345cm}{1.426cm}}
\pgfpathcurveto{\pgfqpoint{1.345cm}{1.463cm}}{\pgfqpoint{1.331cm}{1.497cm}}{\pgfqpoint{1.305cm}{1.523cm}}
\pgfpathcurveto{\pgfqpoint{1.28cm}{1.549cm}}{\pgfqpoint{1.245cm}{1.563cm}}{\pgfqpoint{1.209cm}{1.563cm}}
\pgfpathcurveto{\pgfqpoint{1.172cm}{1.563cm}}{\pgfqpoint{1.138cm}{1.549cm}}{\pgfqpoint{1.112cm}{1.523cm}}
\pgfpathcurveto{\pgfqpoint{1.087cm}{1.497cm}}{\pgfqpoint{1.072cm}{1.463cm}}{\pgfqpoint{1.072cm}{1.426cm}}
\pgfpathcurveto{\pgfqpoint{1.072cm}{1.39cm}}{\pgfqpoint{1.087cm}{1.355cm}}{\pgfqpoint{1.112cm}{1.329cm}}
\pgfpathcurveto{\pgfqpoint{1.138cm}{1.304cm}}{\pgfqpoint{1.172cm}{1.289cm}}{\pgfqpoint{1.209cm}{1.289cm}}
\pgfpathcurveto{\pgfqpoint{1.245cm}{1.289cm}}{\pgfqpoint{1.28cm}{1.304cm}}{\pgfqpoint{1.305cm}{1.329cm}}
\pgfpathcurveto{\pgfqpoint{1.331cm}{1.355cm}}{\pgfqpoint{1.345cm}{1.39cm}}{\pgfqpoint{1.345cm}{1.426cm}}
\pgfusepath{fill}
\begin{pgfscope}
\pgfsetdash{}{0cm}
\pgfsetlinewidth{0.818mm}
\pgfsetroundcap
\pgfsetmiterlimit{4.0}
\pgfpathmoveto{\pgfqpoint{0.682cm}{0.726cm}}
\pgfpathlineto{\pgfqpoint{0.682cm}{0.097cm}}
\pgfusepath{stroke}
\end{pgfscope}
\end{pgfscope}
\end{pgfscope}
\end{pgfscope}
\end{tikzpicture}}} \succ
    \phi_{{M, \varepsilon} }) + 3\lambda^2 b_{M, \varepsilon} \phi_{M, \varepsilon} 
    \\ & \quad + 3\lambda \llbracket X_{M,
    \varepsilon}^2 \rrbracket \circ \chi_{M, \varepsilon}\\
    & =  -\lambda^2 \widetilde{X_{}}_{M, \varepsilon}^{\!\resizebox{!}{.8em}{
\begin{tikzpicture}
\pgfpathmoveto{\pgfqpoint{0cm}{-0.035cm}}
\pgfpathlineto{\pgfqpoint{1.976cm}{-0.035cm}}
\pgfpathlineto{\pgfqpoint{1.976cm}{1.94cm}}
\pgfpathlineto{\pgfqpoint{0cm}{1.94cm}}
\pgfpathclose
\pgfusepath{clip}
\begin{pgfscope}
\begin{pgfscope}
\pgfpathmoveto{\pgfqpoint{0cm}{-0.035cm}}
\pgfpathlineto{\pgfqpoint{1.976cm}{-0.035cm}}
\pgfpathlineto{\pgfqpoint{1.976cm}{1.94cm}}
\pgfpathlineto{\pgfqpoint{0cm}{1.94cm}}
\pgfpathclose
\pgfusepath{clip}
\begin{pgfscope}
\begin{pgfscope}
\pgfsetdash{}{0cm}
\pgfsetlinewidth{0.818mm}
\pgfsetroundcap
\pgfsetroundjoin
\pgfsetmiterlimit{7.0}
\definecolor{eps2pgf_color}{gray}{0}\pgfsetstrokecolor{eps2pgf_color}\pgfsetfillcolor{eps2pgf_color}
\pgfpathmoveto{\pgfqpoint{0.117cm}{1.815cm}}
\pgfpathlineto{\pgfqpoint{0.682cm}{1.065cm}}
\pgfpathlineto{\pgfqpoint{1.246cm}{1.815cm}}
\pgfusepath{stroke}
\end{pgfscope}
\definecolor{eps2pgf_color}{gray}{0}\pgfsetstrokecolor{eps2pgf_color}\pgfsetfillcolor{eps2pgf_color}
\pgfpathmoveto{\pgfqpoint{0.273cm}{1.789cm}}
\pgfpathcurveto{\pgfqpoint{0.273cm}{1.825cm}}{\pgfqpoint{0.259cm}{1.86cm}}{\pgfqpoint{0.233cm}{1.886cm}}
\pgfpathcurveto{\pgfqpoint{0.207cm}{1.912cm}}{\pgfqpoint{0.173cm}{1.926cm}}{\pgfqpoint{0.137cm}{1.926cm}}
\pgfpathcurveto{\pgfqpoint{0.1cm}{1.926cm}}{\pgfqpoint{0.066cm}{1.912cm}}{\pgfqpoint{0.04cm}{1.886cm}}
\pgfpathcurveto{\pgfqpoint{0.014cm}{1.86cm}}{\pgfqpoint{0cm}{1.825cm}}{\pgfqpoint{0cm}{1.789cm}}
\pgfpathcurveto{\pgfqpoint{0cm}{1.753cm}}{\pgfqpoint{0.014cm}{1.718cm}}{\pgfqpoint{0.04cm}{1.692cm}}
\pgfpathcurveto{\pgfqpoint{0.066cm}{1.667cm}}{\pgfqpoint{0.1cm}{1.652cm}}{\pgfqpoint{0.137cm}{1.652cm}}
\pgfpathcurveto{\pgfqpoint{0.173cm}{1.652cm}}{\pgfqpoint{0.207cm}{1.667cm}}{\pgfqpoint{0.233cm}{1.692cm}}
\pgfpathcurveto{\pgfqpoint{0.259cm}{1.718cm}}{\pgfqpoint{0.273cm}{1.753cm}}{\pgfqpoint{0.273cm}{1.789cm}}
\pgfusepath{fill}
\pgfpathmoveto{\pgfqpoint{1.345cm}{1.765cm}}
\pgfpathcurveto{\pgfqpoint{1.345cm}{1.801cm}}{\pgfqpoint{1.331cm}{1.836cm}}{\pgfqpoint{1.305cm}{1.862cm}}
\pgfpathcurveto{\pgfqpoint{1.28cm}{1.887cm}}{\pgfqpoint{1.245cm}{1.902cm}}{\pgfqpoint{1.209cm}{1.902cm}}
\pgfpathcurveto{\pgfqpoint{1.172cm}{1.902cm}}{\pgfqpoint{1.138cm}{1.887cm}}{\pgfqpoint{1.112cm}{1.862cm}}
\pgfpathcurveto{\pgfqpoint{1.087cm}{1.836cm}}{\pgfqpoint{1.072cm}{1.801cm}}{\pgfqpoint{1.072cm}{1.765cm}}
\pgfpathcurveto{\pgfqpoint{1.072cm}{1.728cm}}{\pgfqpoint{1.087cm}{1.694cm}}{\pgfqpoint{1.112cm}{1.668cm}}
\pgfpathcurveto{\pgfqpoint{1.138cm}{1.642cm}}{\pgfqpoint{1.172cm}{1.628cm}}{\pgfqpoint{1.209cm}{1.628cm}}
\pgfpathcurveto{\pgfqpoint{1.245cm}{1.628cm}}{\pgfqpoint{1.28cm}{1.642cm}}{\pgfqpoint{1.305cm}{1.668cm}}
\pgfpathcurveto{\pgfqpoint{1.331cm}{1.694cm}}{\pgfqpoint{1.345cm}{1.728cm}}{\pgfqpoint{1.345cm}{1.765cm}}
\pgfusepath{fill}
\begin{pgfscope}
\pgfsetdash{}{0cm}
\pgfsetlinewidth{0.818mm}
\pgfsetroundcap
\pgfsetroundjoin
\pgfsetmiterlimit{7.0}
\pgfpathmoveto{\pgfqpoint{0.682cm}{1.065cm}}
\pgfpathlineto{\pgfqpoint{1.246cm}{0.315cm}}
\pgfpathlineto{\pgfqpoint{1.811cm}{1.065cm}}
\pgfusepath{stroke}
\end{pgfscope}
\pgfpathmoveto{\pgfqpoint{1.948cm}{1.065cm}}
\pgfpathcurveto{\pgfqpoint{1.948cm}{1.101cm}}{\pgfqpoint{1.933cm}{1.136cm}}{\pgfqpoint{1.907cm}{1.162cm}}
\pgfpathcurveto{\pgfqpoint{1.882cm}{1.187cm}}{\pgfqpoint{1.847cm}{1.202cm}}{\pgfqpoint{1.811cm}{1.202cm}}
\pgfpathcurveto{\pgfqpoint{1.775cm}{1.202cm}}{\pgfqpoint{1.74cm}{1.187cm}}{\pgfqpoint{1.714cm}{1.162cm}}
\pgfpathcurveto{\pgfqpoint{1.689cm}{1.136cm}}{\pgfqpoint{1.674cm}{1.101cm}}{\pgfqpoint{1.674cm}{1.065cm}}
\pgfpathcurveto{\pgfqpoint{1.674cm}{1.029cm}}{\pgfqpoint{1.689cm}{0.994cm}}{\pgfqpoint{1.714cm}{0.968cm}}
\pgfpathcurveto{\pgfqpoint{1.74cm}{0.942cm}}{\pgfqpoint{1.775cm}{0.928cm}}{\pgfqpoint{1.811cm}{0.928cm}}
\pgfpathcurveto{\pgfqpoint{1.847cm}{0.928cm}}{\pgfqpoint{1.882cm}{0.942cm}}{\pgfqpoint{1.907cm}{0.968cm}}
\pgfpathcurveto{\pgfqpoint{1.933cm}{0.994cm}}{\pgfqpoint{1.948cm}{1.029cm}}{\pgfqpoint{1.948cm}{1.065cm}}
\pgfusepath{fill}
\begin{pgfscope}
\pgfsetdash{}{0cm}
\pgfsetlinewidth{0.818mm}
\pgfsetmiterlimit{7.0}
\pgfpathmoveto{\pgfqpoint{1.246cm}{0.315cm}}
\pgfpathlineto{\pgfqpoint{1.244cm}{1.061cm}}
\pgfusepath{stroke}
\end{pgfscope}
\pgfpathmoveto{\pgfqpoint{1.38cm}{1.065cm}}
\pgfpathcurveto{\pgfqpoint{1.38cm}{1.101cm}}{\pgfqpoint{1.366cm}{1.136cm}}{\pgfqpoint{1.34cm}{1.162cm}}
\pgfpathcurveto{\pgfqpoint{1.315cm}{1.187cm}}{\pgfqpoint{1.28cm}{1.202cm}}{\pgfqpoint{1.244cm}{1.202cm}}
\pgfpathcurveto{\pgfqpoint{1.207cm}{1.202cm}}{\pgfqpoint{1.173cm}{1.187cm}}{\pgfqpoint{1.147cm}{1.162cm}}
\pgfpathcurveto{\pgfqpoint{1.121cm}{1.136cm}}{\pgfqpoint{1.107cm}{1.101cm}}{\pgfqpoint{1.107cm}{1.065cm}}
\pgfpathcurveto{\pgfqpoint{1.107cm}{1.029cm}}{\pgfqpoint{1.121cm}{0.994cm}}{\pgfqpoint{1.147cm}{0.968cm}}
\pgfpathcurveto{\pgfqpoint{1.173cm}{0.942cm}}{\pgfqpoint{1.207cm}{0.928cm}}{\pgfqpoint{1.244cm}{0.928cm}}
\pgfpathcurveto{\pgfqpoint{1.28cm}{0.928cm}}{\pgfqpoint{1.315cm}{0.942cm}}{\pgfqpoint{1.34cm}{0.968cm}}
\pgfpathcurveto{\pgfqpoint{1.366cm}{0.994cm}}{\pgfqpoint{1.38cm}{1.029cm}}{\pgfqpoint{1.38cm}{1.065cm}}
\pgfusepath{fill}
\begin{pgfscope}
\pgfsetdash{}{0cm}
\pgfsetlinewidth{0.818mm}
\pgfsetmiterlimit{4.0}
\pgfpathmoveto{\pgfqpoint{1.383cm}{0.178cm}}
\pgfpathcurveto{\pgfqpoint{1.383cm}{0.214cm}}{\pgfqpoint{1.369cm}{0.249cm}}{\pgfqpoint{1.343cm}{0.275cm}}
\pgfpathcurveto{\pgfqpoint{1.317cm}{0.3cm}}{\pgfqpoint{1.283cm}{0.315cm}}{\pgfqpoint{1.246cm}{0.315cm}}
\pgfpathcurveto{\pgfqpoint{1.21cm}{0.315cm}}{\pgfqpoint{1.175cm}{0.3cm}}{\pgfqpoint{1.15cm}{0.275cm}}
\pgfpathcurveto{\pgfqpoint{1.124cm}{0.249cm}}{\pgfqpoint{1.11cm}{0.214cm}}{\pgfqpoint{1.11cm}{0.178cm}}
\pgfpathcurveto{\pgfqpoint{1.11cm}{0.141cm}}{\pgfqpoint{1.124cm}{0.107cm}}{\pgfqpoint{1.15cm}{0.081cm}}
\pgfpathcurveto{\pgfqpoint{1.175cm}{0.055cm}}{\pgfqpoint{1.21cm}{0.041cm}}{\pgfqpoint{1.246cm}{0.041cm}}
\pgfpathcurveto{\pgfqpoint{1.283cm}{0.041cm}}{\pgfqpoint{1.317cm}{0.055cm}}{\pgfqpoint{1.343cm}{0.081cm}}
\pgfpathcurveto{\pgfqpoint{1.369cm}{0.107cm}}{\pgfqpoint{1.383cm}{0.141cm}}{\pgfqpoint{1.383cm}{0.178cm}}
\pgfusepath{stroke}
\end{pgfscope}
\end{pgfscope}
\end{pgfscope}
\end{pgfscope}
\end{tikzpicture}}} \phi_{M,
    \varepsilon} + 3\lambda^2 (b_{M, \varepsilon} - \tilde{b}_{M, \varepsilon} (t))
    \phi_{M, \varepsilon}\\
    &  \quad +\lambda^2 C_{\varepsilon} (\phi_{M, \varepsilon}, - 3 X_{M,
    \varepsilon}^{\!\resizebox{0.6em}{!}{
\begin{tikzpicture}
\pgfpathmoveto{\pgfqpoint{0cm}{0cm}}
\pgfpathlineto{\pgfqpoint{1.376cm}{0cm}}
\pgfpathlineto{\pgfqpoint{1.376cm}{1.588cm}}
\pgfpathlineto{\pgfqpoint{0cm}{1.588cm}}
\pgfpathclose
\pgfusepath{clip}
\begin{pgfscope}
\begin{pgfscope}
\pgfpathmoveto{\pgfqpoint{0cm}{0cm}}
\pgfpathlineto{\pgfqpoint{1.376cm}{0cm}}
\pgfpathlineto{\pgfqpoint{1.376cm}{1.588cm}}
\pgfpathlineto{\pgfqpoint{0cm}{1.588cm}}
\pgfpathclose
\pgfusepath{clip}
\begin{pgfscope}
\begin{pgfscope}
\definecolor{eps2pgf_color}{gray}{0.976471}\pgfsetstrokecolor{eps2pgf_color}\pgfsetfillcolor{eps2pgf_color}
\pgfpathmoveto{\pgfqpoint{0cm}{0cm}}
\pgfpathlineto{\pgfqpoint{1.376cm}{0cm}}
\pgfpathlineto{\pgfqpoint{1.376cm}{1.588cm}}
\pgfpathlineto{\pgfqpoint{0cm}{1.588cm}}
\pgfpathclose
\pgfusepath{fill}
\end{pgfscope}
\begin{pgfscope}
\pgfsetdash{}{0cm}
\pgfsetlinewidth{0.818mm}
\pgfsetroundcap
\pgfsetroundjoin
\pgfsetmiterlimit{7.0}
\definecolor{eps2pgf_color}{gray}{0}\pgfsetstrokecolor{eps2pgf_color}\pgfsetfillcolor{eps2pgf_color}
\pgfpathmoveto{\pgfqpoint{0.117cm}{1.476cm}}
\pgfpathlineto{\pgfqpoint{0.682cm}{0.726cm}}
\pgfpathlineto{\pgfqpoint{1.246cm}{1.476cm}}
\pgfusepath{stroke}
\end{pgfscope}
\definecolor{eps2pgf_color}{gray}{0}\pgfsetstrokecolor{eps2pgf_color}\pgfsetfillcolor{eps2pgf_color}
\pgfpathmoveto{\pgfqpoint{0.273cm}{1.451cm}}
\pgfpathcurveto{\pgfqpoint{0.273cm}{1.487cm}}{\pgfqpoint{0.259cm}{1.522cm}}{\pgfqpoint{0.233cm}{1.547cm}}
\pgfpathcurveto{\pgfqpoint{0.207cm}{1.573cm}}{\pgfqpoint{0.173cm}{1.588cm}}{\pgfqpoint{0.137cm}{1.588cm}}
\pgfpathcurveto{\pgfqpoint{0.1cm}{1.588cm}}{\pgfqpoint{0.066cm}{1.573cm}}{\pgfqpoint{0.04cm}{1.547cm}}
\pgfpathcurveto{\pgfqpoint{0.014cm}{1.522cm}}{\pgfqpoint{0cm}{1.487cm}}{\pgfqpoint{0cm}{1.451cm}}
\pgfpathcurveto{\pgfqpoint{0cm}{1.414cm}}{\pgfqpoint{0.014cm}{1.379cm}}{\pgfqpoint{0.04cm}{1.354cm}}
\pgfpathcurveto{\pgfqpoint{0.066cm}{1.328cm}}{\pgfqpoint{0.1cm}{1.314cm}}{\pgfqpoint{0.137cm}{1.314cm}}
\pgfpathcurveto{\pgfqpoint{0.173cm}{1.314cm}}{\pgfqpoint{0.207cm}{1.328cm}}{\pgfqpoint{0.233cm}{1.354cm}}
\pgfpathcurveto{\pgfqpoint{0.259cm}{1.379cm}}{\pgfqpoint{0.273cm}{1.414cm}}{\pgfqpoint{0.273cm}{1.451cm}}
\pgfusepath{fill}
\pgfpathmoveto{\pgfqpoint{1.345cm}{1.426cm}}
\pgfpathcurveto{\pgfqpoint{1.345cm}{1.463cm}}{\pgfqpoint{1.331cm}{1.497cm}}{\pgfqpoint{1.305cm}{1.523cm}}
\pgfpathcurveto{\pgfqpoint{1.28cm}{1.549cm}}{\pgfqpoint{1.245cm}{1.563cm}}{\pgfqpoint{1.209cm}{1.563cm}}
\pgfpathcurveto{\pgfqpoint{1.172cm}{1.563cm}}{\pgfqpoint{1.138cm}{1.549cm}}{\pgfqpoint{1.112cm}{1.523cm}}
\pgfpathcurveto{\pgfqpoint{1.087cm}{1.497cm}}{\pgfqpoint{1.072cm}{1.463cm}}{\pgfqpoint{1.072cm}{1.426cm}}
\pgfpathcurveto{\pgfqpoint{1.072cm}{1.39cm}}{\pgfqpoint{1.087cm}{1.355cm}}{\pgfqpoint{1.112cm}{1.329cm}}
\pgfpathcurveto{\pgfqpoint{1.138cm}{1.304cm}}{\pgfqpoint{1.172cm}{1.289cm}}{\pgfqpoint{1.209cm}{1.289cm}}
\pgfpathcurveto{\pgfqpoint{1.245cm}{1.289cm}}{\pgfqpoint{1.28cm}{1.304cm}}{\pgfqpoint{1.305cm}{1.329cm}}
\pgfpathcurveto{\pgfqpoint{1.331cm}{1.355cm}}{\pgfqpoint{1.345cm}{1.39cm}}{\pgfqpoint{1.345cm}{1.426cm}}
\pgfusepath{fill}
\begin{pgfscope}
\pgfsetdash{}{0cm}
\pgfsetlinewidth{0.818mm}
\pgfsetroundcap
\pgfsetmiterlimit{4.0}
\pgfpathmoveto{\pgfqpoint{0.682cm}{0.726cm}}
\pgfpathlineto{\pgfqpoint{0.682cm}{0.097cm}}
\pgfusepath{stroke}
\end{pgfscope}
\end{pgfscope}
\end{pgfscope}
\end{pgfscope}
\end{tikzpicture}}}, 3 \llbracket X_{M, \varepsilon}^2 \rrbracket) + 3\lambda
    \llbracket X_{M, \varepsilon}^2 \rrbracket \circ \chi_{M, \varepsilon} .
  \end{aligned} 
\end{equation*}
Recalling that $Z_{M, \varepsilon} = - 3\lambda^{-1} \llbracket X_{M, \varepsilon}^2
\rrbracket \circ Y_{M, \varepsilon} - 3 b_{M, \varepsilon} (X_{M, \varepsilon}
+ Y_{M, \varepsilon})$ can be rewritten as {\eqref{eq:def-Z}} and controlled
due to Lemma~\ref{lem:Z}, where we also estimated $X_{M, \varepsilon} Y_{M,
\varepsilon}$ and $X_{M, \varepsilon} Y^2_{M, \varepsilon}$, we deduce
\[ \begin{array}{lll}
     U_{M, \varepsilon} & = & -\lambda^2 \widetilde{X_{}}_{M, \varepsilon}^{\!\resizebox{!}{.8em}{
\begin{tikzpicture}
\pgfpathmoveto{\pgfqpoint{0cm}{-0.035cm}}
\pgfpathlineto{\pgfqpoint{1.976cm}{-0.035cm}}
\pgfpathlineto{\pgfqpoint{1.976cm}{1.94cm}}
\pgfpathlineto{\pgfqpoint{0cm}{1.94cm}}
\pgfpathclose
\pgfusepath{clip}
\begin{pgfscope}
\begin{pgfscope}
\pgfpathmoveto{\pgfqpoint{0cm}{-0.035cm}}
\pgfpathlineto{\pgfqpoint{1.976cm}{-0.035cm}}
\pgfpathlineto{\pgfqpoint{1.976cm}{1.94cm}}
\pgfpathlineto{\pgfqpoint{0cm}{1.94cm}}
\pgfpathclose
\pgfusepath{clip}
\begin{pgfscope}
\begin{pgfscope}
\pgfsetdash{}{0cm}
\pgfsetlinewidth{0.818mm}
\pgfsetroundcap
\pgfsetroundjoin
\pgfsetmiterlimit{7.0}
\definecolor{eps2pgf_color}{gray}{0}\pgfsetstrokecolor{eps2pgf_color}\pgfsetfillcolor{eps2pgf_color}
\pgfpathmoveto{\pgfqpoint{0.117cm}{1.815cm}}
\pgfpathlineto{\pgfqpoint{0.682cm}{1.065cm}}
\pgfpathlineto{\pgfqpoint{1.246cm}{1.815cm}}
\pgfusepath{stroke}
\end{pgfscope}
\definecolor{eps2pgf_color}{gray}{0}\pgfsetstrokecolor{eps2pgf_color}\pgfsetfillcolor{eps2pgf_color}
\pgfpathmoveto{\pgfqpoint{0.273cm}{1.789cm}}
\pgfpathcurveto{\pgfqpoint{0.273cm}{1.825cm}}{\pgfqpoint{0.259cm}{1.86cm}}{\pgfqpoint{0.233cm}{1.886cm}}
\pgfpathcurveto{\pgfqpoint{0.207cm}{1.912cm}}{\pgfqpoint{0.173cm}{1.926cm}}{\pgfqpoint{0.137cm}{1.926cm}}
\pgfpathcurveto{\pgfqpoint{0.1cm}{1.926cm}}{\pgfqpoint{0.066cm}{1.912cm}}{\pgfqpoint{0.04cm}{1.886cm}}
\pgfpathcurveto{\pgfqpoint{0.014cm}{1.86cm}}{\pgfqpoint{0cm}{1.825cm}}{\pgfqpoint{0cm}{1.789cm}}
\pgfpathcurveto{\pgfqpoint{0cm}{1.753cm}}{\pgfqpoint{0.014cm}{1.718cm}}{\pgfqpoint{0.04cm}{1.692cm}}
\pgfpathcurveto{\pgfqpoint{0.066cm}{1.667cm}}{\pgfqpoint{0.1cm}{1.652cm}}{\pgfqpoint{0.137cm}{1.652cm}}
\pgfpathcurveto{\pgfqpoint{0.173cm}{1.652cm}}{\pgfqpoint{0.207cm}{1.667cm}}{\pgfqpoint{0.233cm}{1.692cm}}
\pgfpathcurveto{\pgfqpoint{0.259cm}{1.718cm}}{\pgfqpoint{0.273cm}{1.753cm}}{\pgfqpoint{0.273cm}{1.789cm}}
\pgfusepath{fill}
\pgfpathmoveto{\pgfqpoint{1.345cm}{1.765cm}}
\pgfpathcurveto{\pgfqpoint{1.345cm}{1.801cm}}{\pgfqpoint{1.331cm}{1.836cm}}{\pgfqpoint{1.305cm}{1.862cm}}
\pgfpathcurveto{\pgfqpoint{1.28cm}{1.887cm}}{\pgfqpoint{1.245cm}{1.902cm}}{\pgfqpoint{1.209cm}{1.902cm}}
\pgfpathcurveto{\pgfqpoint{1.172cm}{1.902cm}}{\pgfqpoint{1.138cm}{1.887cm}}{\pgfqpoint{1.112cm}{1.862cm}}
\pgfpathcurveto{\pgfqpoint{1.087cm}{1.836cm}}{\pgfqpoint{1.072cm}{1.801cm}}{\pgfqpoint{1.072cm}{1.765cm}}
\pgfpathcurveto{\pgfqpoint{1.072cm}{1.728cm}}{\pgfqpoint{1.087cm}{1.694cm}}{\pgfqpoint{1.112cm}{1.668cm}}
\pgfpathcurveto{\pgfqpoint{1.138cm}{1.642cm}}{\pgfqpoint{1.172cm}{1.628cm}}{\pgfqpoint{1.209cm}{1.628cm}}
\pgfpathcurveto{\pgfqpoint{1.245cm}{1.628cm}}{\pgfqpoint{1.28cm}{1.642cm}}{\pgfqpoint{1.305cm}{1.668cm}}
\pgfpathcurveto{\pgfqpoint{1.331cm}{1.694cm}}{\pgfqpoint{1.345cm}{1.728cm}}{\pgfqpoint{1.345cm}{1.765cm}}
\pgfusepath{fill}
\begin{pgfscope}
\pgfsetdash{}{0cm}
\pgfsetlinewidth{0.818mm}
\pgfsetroundcap
\pgfsetroundjoin
\pgfsetmiterlimit{7.0}
\pgfpathmoveto{\pgfqpoint{0.682cm}{1.065cm}}
\pgfpathlineto{\pgfqpoint{1.246cm}{0.315cm}}
\pgfpathlineto{\pgfqpoint{1.811cm}{1.065cm}}
\pgfusepath{stroke}
\end{pgfscope}
\pgfpathmoveto{\pgfqpoint{1.948cm}{1.065cm}}
\pgfpathcurveto{\pgfqpoint{1.948cm}{1.101cm}}{\pgfqpoint{1.933cm}{1.136cm}}{\pgfqpoint{1.907cm}{1.162cm}}
\pgfpathcurveto{\pgfqpoint{1.882cm}{1.187cm}}{\pgfqpoint{1.847cm}{1.202cm}}{\pgfqpoint{1.811cm}{1.202cm}}
\pgfpathcurveto{\pgfqpoint{1.775cm}{1.202cm}}{\pgfqpoint{1.74cm}{1.187cm}}{\pgfqpoint{1.714cm}{1.162cm}}
\pgfpathcurveto{\pgfqpoint{1.689cm}{1.136cm}}{\pgfqpoint{1.674cm}{1.101cm}}{\pgfqpoint{1.674cm}{1.065cm}}
\pgfpathcurveto{\pgfqpoint{1.674cm}{1.029cm}}{\pgfqpoint{1.689cm}{0.994cm}}{\pgfqpoint{1.714cm}{0.968cm}}
\pgfpathcurveto{\pgfqpoint{1.74cm}{0.942cm}}{\pgfqpoint{1.775cm}{0.928cm}}{\pgfqpoint{1.811cm}{0.928cm}}
\pgfpathcurveto{\pgfqpoint{1.847cm}{0.928cm}}{\pgfqpoint{1.882cm}{0.942cm}}{\pgfqpoint{1.907cm}{0.968cm}}
\pgfpathcurveto{\pgfqpoint{1.933cm}{0.994cm}}{\pgfqpoint{1.948cm}{1.029cm}}{\pgfqpoint{1.948cm}{1.065cm}}
\pgfusepath{fill}
\begin{pgfscope}
\pgfsetdash{}{0cm}
\pgfsetlinewidth{0.818mm}
\pgfsetmiterlimit{7.0}
\pgfpathmoveto{\pgfqpoint{1.246cm}{0.315cm}}
\pgfpathlineto{\pgfqpoint{1.244cm}{1.061cm}}
\pgfusepath{stroke}
\end{pgfscope}
\pgfpathmoveto{\pgfqpoint{1.38cm}{1.065cm}}
\pgfpathcurveto{\pgfqpoint{1.38cm}{1.101cm}}{\pgfqpoint{1.366cm}{1.136cm}}{\pgfqpoint{1.34cm}{1.162cm}}
\pgfpathcurveto{\pgfqpoint{1.315cm}{1.187cm}}{\pgfqpoint{1.28cm}{1.202cm}}{\pgfqpoint{1.244cm}{1.202cm}}
\pgfpathcurveto{\pgfqpoint{1.207cm}{1.202cm}}{\pgfqpoint{1.173cm}{1.187cm}}{\pgfqpoint{1.147cm}{1.162cm}}
\pgfpathcurveto{\pgfqpoint{1.121cm}{1.136cm}}{\pgfqpoint{1.107cm}{1.101cm}}{\pgfqpoint{1.107cm}{1.065cm}}
\pgfpathcurveto{\pgfqpoint{1.107cm}{1.029cm}}{\pgfqpoint{1.121cm}{0.994cm}}{\pgfqpoint{1.147cm}{0.968cm}}
\pgfpathcurveto{\pgfqpoint{1.173cm}{0.942cm}}{\pgfqpoint{1.207cm}{0.928cm}}{\pgfqpoint{1.244cm}{0.928cm}}
\pgfpathcurveto{\pgfqpoint{1.28cm}{0.928cm}}{\pgfqpoint{1.315cm}{0.942cm}}{\pgfqpoint{1.34cm}{0.968cm}}
\pgfpathcurveto{\pgfqpoint{1.366cm}{0.994cm}}{\pgfqpoint{1.38cm}{1.029cm}}{\pgfqpoint{1.38cm}{1.065cm}}
\pgfusepath{fill}
\begin{pgfscope}
\pgfsetdash{}{0cm}
\pgfsetlinewidth{0.818mm}
\pgfsetmiterlimit{4.0}
\pgfpathmoveto{\pgfqpoint{1.383cm}{0.178cm}}
\pgfpathcurveto{\pgfqpoint{1.383cm}{0.214cm}}{\pgfqpoint{1.369cm}{0.249cm}}{\pgfqpoint{1.343cm}{0.275cm}}
\pgfpathcurveto{\pgfqpoint{1.317cm}{0.3cm}}{\pgfqpoint{1.283cm}{0.315cm}}{\pgfqpoint{1.246cm}{0.315cm}}
\pgfpathcurveto{\pgfqpoint{1.21cm}{0.315cm}}{\pgfqpoint{1.175cm}{0.3cm}}{\pgfqpoint{1.15cm}{0.275cm}}
\pgfpathcurveto{\pgfqpoint{1.124cm}{0.249cm}}{\pgfqpoint{1.11cm}{0.214cm}}{\pgfqpoint{1.11cm}{0.178cm}}
\pgfpathcurveto{\pgfqpoint{1.11cm}{0.141cm}}{\pgfqpoint{1.124cm}{0.107cm}}{\pgfqpoint{1.15cm}{0.081cm}}
\pgfpathcurveto{\pgfqpoint{1.175cm}{0.055cm}}{\pgfqpoint{1.21cm}{0.041cm}}{\pgfqpoint{1.246cm}{0.041cm}}
\pgfpathcurveto{\pgfqpoint{1.283cm}{0.041cm}}{\pgfqpoint{1.317cm}{0.055cm}}{\pgfqpoint{1.343cm}{0.081cm}}
\pgfpathcurveto{\pgfqpoint{1.369cm}{0.107cm}}{\pgfqpoint{1.383cm}{0.141cm}}{\pgfqpoint{1.383cm}{0.178cm}}
\pgfusepath{stroke}
\end{pgfscope}
\end{pgfscope}
\end{pgfscope}
\end{pgfscope}
\end{tikzpicture}}}
     \phi_{M, \varepsilon} + 3\lambda^2 (b_{M, \varepsilon} - \tilde{b}_{M,
     \varepsilon} (t)) \phi_{M, \varepsilon} +\lambda^2  C_{\varepsilon} (\phi_{M,
     \varepsilon}, - 3 X_{M, \varepsilon}^{\!\resizebox{0.6em}{!}{
\begin{tikzpicture}
\pgfpathmoveto{\pgfqpoint{0cm}{0cm}}
\pgfpathlineto{\pgfqpoint{1.376cm}{0cm}}
\pgfpathlineto{\pgfqpoint{1.376cm}{1.588cm}}
\pgfpathlineto{\pgfqpoint{0cm}{1.588cm}}
\pgfpathclose
\pgfusepath{clip}
\begin{pgfscope}
\begin{pgfscope}
\pgfpathmoveto{\pgfqpoint{0cm}{0cm}}
\pgfpathlineto{\pgfqpoint{1.376cm}{0cm}}
\pgfpathlineto{\pgfqpoint{1.376cm}{1.588cm}}
\pgfpathlineto{\pgfqpoint{0cm}{1.588cm}}
\pgfpathclose
\pgfusepath{clip}
\begin{pgfscope}
\begin{pgfscope}
\definecolor{eps2pgf_color}{gray}{0.976471}\pgfsetstrokecolor{eps2pgf_color}\pgfsetfillcolor{eps2pgf_color}
\pgfpathmoveto{\pgfqpoint{0cm}{0cm}}
\pgfpathlineto{\pgfqpoint{1.376cm}{0cm}}
\pgfpathlineto{\pgfqpoint{1.376cm}{1.588cm}}
\pgfpathlineto{\pgfqpoint{0cm}{1.588cm}}
\pgfpathclose
\pgfusepath{fill}
\end{pgfscope}
\begin{pgfscope}
\pgfsetdash{}{0cm}
\pgfsetlinewidth{0.818mm}
\pgfsetroundcap
\pgfsetroundjoin
\pgfsetmiterlimit{7.0}
\definecolor{eps2pgf_color}{gray}{0}\pgfsetstrokecolor{eps2pgf_color}\pgfsetfillcolor{eps2pgf_color}
\pgfpathmoveto{\pgfqpoint{0.117cm}{1.476cm}}
\pgfpathlineto{\pgfqpoint{0.682cm}{0.726cm}}
\pgfpathlineto{\pgfqpoint{1.246cm}{1.476cm}}
\pgfusepath{stroke}
\end{pgfscope}
\definecolor{eps2pgf_color}{gray}{0}\pgfsetstrokecolor{eps2pgf_color}\pgfsetfillcolor{eps2pgf_color}
\pgfpathmoveto{\pgfqpoint{0.273cm}{1.451cm}}
\pgfpathcurveto{\pgfqpoint{0.273cm}{1.487cm}}{\pgfqpoint{0.259cm}{1.522cm}}{\pgfqpoint{0.233cm}{1.547cm}}
\pgfpathcurveto{\pgfqpoint{0.207cm}{1.573cm}}{\pgfqpoint{0.173cm}{1.588cm}}{\pgfqpoint{0.137cm}{1.588cm}}
\pgfpathcurveto{\pgfqpoint{0.1cm}{1.588cm}}{\pgfqpoint{0.066cm}{1.573cm}}{\pgfqpoint{0.04cm}{1.547cm}}
\pgfpathcurveto{\pgfqpoint{0.014cm}{1.522cm}}{\pgfqpoint{0cm}{1.487cm}}{\pgfqpoint{0cm}{1.451cm}}
\pgfpathcurveto{\pgfqpoint{0cm}{1.414cm}}{\pgfqpoint{0.014cm}{1.379cm}}{\pgfqpoint{0.04cm}{1.354cm}}
\pgfpathcurveto{\pgfqpoint{0.066cm}{1.328cm}}{\pgfqpoint{0.1cm}{1.314cm}}{\pgfqpoint{0.137cm}{1.314cm}}
\pgfpathcurveto{\pgfqpoint{0.173cm}{1.314cm}}{\pgfqpoint{0.207cm}{1.328cm}}{\pgfqpoint{0.233cm}{1.354cm}}
\pgfpathcurveto{\pgfqpoint{0.259cm}{1.379cm}}{\pgfqpoint{0.273cm}{1.414cm}}{\pgfqpoint{0.273cm}{1.451cm}}
\pgfusepath{fill}
\pgfpathmoveto{\pgfqpoint{1.345cm}{1.426cm}}
\pgfpathcurveto{\pgfqpoint{1.345cm}{1.463cm}}{\pgfqpoint{1.331cm}{1.497cm}}{\pgfqpoint{1.305cm}{1.523cm}}
\pgfpathcurveto{\pgfqpoint{1.28cm}{1.549cm}}{\pgfqpoint{1.245cm}{1.563cm}}{\pgfqpoint{1.209cm}{1.563cm}}
\pgfpathcurveto{\pgfqpoint{1.172cm}{1.563cm}}{\pgfqpoint{1.138cm}{1.549cm}}{\pgfqpoint{1.112cm}{1.523cm}}
\pgfpathcurveto{\pgfqpoint{1.087cm}{1.497cm}}{\pgfqpoint{1.072cm}{1.463cm}}{\pgfqpoint{1.072cm}{1.426cm}}
\pgfpathcurveto{\pgfqpoint{1.072cm}{1.39cm}}{\pgfqpoint{1.087cm}{1.355cm}}{\pgfqpoint{1.112cm}{1.329cm}}
\pgfpathcurveto{\pgfqpoint{1.138cm}{1.304cm}}{\pgfqpoint{1.172cm}{1.289cm}}{\pgfqpoint{1.209cm}{1.289cm}}
\pgfpathcurveto{\pgfqpoint{1.245cm}{1.289cm}}{\pgfqpoint{1.28cm}{1.304cm}}{\pgfqpoint{1.305cm}{1.329cm}}
\pgfpathcurveto{\pgfqpoint{1.331cm}{1.355cm}}{\pgfqpoint{1.345cm}{1.39cm}}{\pgfqpoint{1.345cm}{1.426cm}}
\pgfusepath{fill}
\begin{pgfscope}
\pgfsetdash{}{0cm}
\pgfsetlinewidth{0.818mm}
\pgfsetroundcap
\pgfsetmiterlimit{4.0}
\pgfpathmoveto{\pgfqpoint{0.682cm}{0.726cm}}
\pgfpathlineto{\pgfqpoint{0.682cm}{0.097cm}}
\pgfusepath{stroke}
\end{pgfscope}
\end{pgfscope}
\end{pgfscope}
\end{pgfscope}
\end{tikzpicture}}}, 3 \llbracket X_{M,
     \varepsilon}^2 \rrbracket) 
     \\ & &  + 3\lambda \llbracket X_{M, \varepsilon}^2 \rrbracket
     \circ \chi_{M, \varepsilon}
     \\
     &  & +\lambda^2 Z_{M, \varepsilon} - 3\lambda \llbracket X_{M, \varepsilon}^2 \rrbracket \prec (Y_{M,
     \varepsilon} + \phi_{M, \varepsilon}) - 3\lambda (
     \UU^{\varepsilon}_{\leqslant} \llbracket X_{M, \varepsilon}^2 \rrbracket
     ) \succ Y_{M, \varepsilon} - 3\lambda X_{M, \varepsilon} Y^2_{M,
     \varepsilon} \\
     &  & - 6\lambda X_{M, \varepsilon} Y_{M, \varepsilon} \phi_{M,
     \varepsilon} - 3\lambda X_{M, \varepsilon} \phi_{M, \varepsilon}^2 -\lambda Y_{M,
     \varepsilon}^3 - 3\lambda Y_{M, \varepsilon}^2 \phi_{M, \varepsilon} - 3\lambda Y_{M,
     \varepsilon} \phi_{M, \varepsilon}^2 -\lambda \phi_{M, \varepsilon}^3 .
   \end{array} \]
Consequently, the equation satisfied by $\chi_{M, \varepsilon}$ reads
\begin{equation}
  \begin{array}{lll}
    \LL_{\varepsilon} \chi_{M, \varepsilon} & = & \LL_{\varepsilon} \phi_{M,
    \varepsilon} + 3\lambda \llbracket X_{M, \varepsilon}^2 \rrbracket \succ \phi_{M,
    \varepsilon} + 3 \lambdaX_{M, \varepsilon}^{\!\resizebox{0.6em}{!}{
\begin{tikzpicture}
\pgfpathmoveto{\pgfqpoint{0cm}{0cm}}
\pgfpathlineto{\pgfqpoint{1.376cm}{0cm}}
\pgfpathlineto{\pgfqpoint{1.376cm}{1.588cm}}
\pgfpathlineto{\pgfqpoint{0cm}{1.588cm}}
\pgfpathclose
\pgfusepath{clip}
\begin{pgfscope}
\begin{pgfscope}
\pgfpathmoveto{\pgfqpoint{0cm}{0cm}}
\pgfpathlineto{\pgfqpoint{1.376cm}{0cm}}
\pgfpathlineto{\pgfqpoint{1.376cm}{1.588cm}}
\pgfpathlineto{\pgfqpoint{0cm}{1.588cm}}
\pgfpathclose
\pgfusepath{clip}
\begin{pgfscope}
\begin{pgfscope}
\definecolor{eps2pgf_color}{gray}{0.976471}\pgfsetstrokecolor{eps2pgf_color}\pgfsetfillcolor{eps2pgf_color}
\pgfpathmoveto{\pgfqpoint{0cm}{0cm}}
\pgfpathlineto{\pgfqpoint{1.376cm}{0cm}}
\pgfpathlineto{\pgfqpoint{1.376cm}{1.588cm}}
\pgfpathlineto{\pgfqpoint{0cm}{1.588cm}}
\pgfpathclose
\pgfusepath{fill}
\end{pgfscope}
\begin{pgfscope}
\pgfsetdash{}{0cm}
\pgfsetlinewidth{0.818mm}
\pgfsetroundcap
\pgfsetroundjoin
\pgfsetmiterlimit{7.0}
\definecolor{eps2pgf_color}{gray}{0}\pgfsetstrokecolor{eps2pgf_color}\pgfsetfillcolor{eps2pgf_color}
\pgfpathmoveto{\pgfqpoint{0.117cm}{1.476cm}}
\pgfpathlineto{\pgfqpoint{0.682cm}{0.726cm}}
\pgfpathlineto{\pgfqpoint{1.246cm}{1.476cm}}
\pgfusepath{stroke}
\end{pgfscope}
\definecolor{eps2pgf_color}{gray}{0}\pgfsetstrokecolor{eps2pgf_color}\pgfsetfillcolor{eps2pgf_color}
\pgfpathmoveto{\pgfqpoint{0.273cm}{1.451cm}}
\pgfpathcurveto{\pgfqpoint{0.273cm}{1.487cm}}{\pgfqpoint{0.259cm}{1.522cm}}{\pgfqpoint{0.233cm}{1.547cm}}
\pgfpathcurveto{\pgfqpoint{0.207cm}{1.573cm}}{\pgfqpoint{0.173cm}{1.588cm}}{\pgfqpoint{0.137cm}{1.588cm}}
\pgfpathcurveto{\pgfqpoint{0.1cm}{1.588cm}}{\pgfqpoint{0.066cm}{1.573cm}}{\pgfqpoint{0.04cm}{1.547cm}}
\pgfpathcurveto{\pgfqpoint{0.014cm}{1.522cm}}{\pgfqpoint{0cm}{1.487cm}}{\pgfqpoint{0cm}{1.451cm}}
\pgfpathcurveto{\pgfqpoint{0cm}{1.414cm}}{\pgfqpoint{0.014cm}{1.379cm}}{\pgfqpoint{0.04cm}{1.354cm}}
\pgfpathcurveto{\pgfqpoint{0.066cm}{1.328cm}}{\pgfqpoint{0.1cm}{1.314cm}}{\pgfqpoint{0.137cm}{1.314cm}}
\pgfpathcurveto{\pgfqpoint{0.173cm}{1.314cm}}{\pgfqpoint{0.207cm}{1.328cm}}{\pgfqpoint{0.233cm}{1.354cm}}
\pgfpathcurveto{\pgfqpoint{0.259cm}{1.379cm}}{\pgfqpoint{0.273cm}{1.414cm}}{\pgfqpoint{0.273cm}{1.451cm}}
\pgfusepath{fill}
\pgfpathmoveto{\pgfqpoint{1.345cm}{1.426cm}}
\pgfpathcurveto{\pgfqpoint{1.345cm}{1.463cm}}{\pgfqpoint{1.331cm}{1.497cm}}{\pgfqpoint{1.305cm}{1.523cm}}
\pgfpathcurveto{\pgfqpoint{1.28cm}{1.549cm}}{\pgfqpoint{1.245cm}{1.563cm}}{\pgfqpoint{1.209cm}{1.563cm}}
\pgfpathcurveto{\pgfqpoint{1.172cm}{1.563cm}}{\pgfqpoint{1.138cm}{1.549cm}}{\pgfqpoint{1.112cm}{1.523cm}}
\pgfpathcurveto{\pgfqpoint{1.087cm}{1.497cm}}{\pgfqpoint{1.072cm}{1.463cm}}{\pgfqpoint{1.072cm}{1.426cm}}
\pgfpathcurveto{\pgfqpoint{1.072cm}{1.39cm}}{\pgfqpoint{1.087cm}{1.355cm}}{\pgfqpoint{1.112cm}{1.329cm}}
\pgfpathcurveto{\pgfqpoint{1.138cm}{1.304cm}}{\pgfqpoint{1.172cm}{1.289cm}}{\pgfqpoint{1.209cm}{1.289cm}}
\pgfpathcurveto{\pgfqpoint{1.245cm}{1.289cm}}{\pgfqpoint{1.28cm}{1.304cm}}{\pgfqpoint{1.305cm}{1.329cm}}
\pgfpathcurveto{\pgfqpoint{1.331cm}{1.355cm}}{\pgfqpoint{1.345cm}{1.39cm}}{\pgfqpoint{1.345cm}{1.426cm}}
\pgfusepath{fill}
\begin{pgfscope}
\pgfsetdash{}{0cm}
\pgfsetlinewidth{0.818mm}
\pgfsetroundcap
\pgfsetmiterlimit{4.0}
\pgfpathmoveto{\pgfqpoint{0.682cm}{0.726cm}}
\pgfpathlineto{\pgfqpoint{0.682cm}{0.097cm}}
\pgfusepath{stroke}
\end{pgfscope}
\end{pgfscope}
\end{pgfscope}
\end{pgfscope}
\end{tikzpicture}}} \succ \LL_{\varepsilon}
    \phi_{M, \varepsilon} - 6\lambda \nabla_{\varepsilon} X_{M,
    \varepsilon}^{\!\resizebox{0.6em}{!}{
\begin{tikzpicture}
\pgfpathmoveto{\pgfqpoint{0cm}{0cm}}
\pgfpathlineto{\pgfqpoint{1.376cm}{0cm}}
\pgfpathlineto{\pgfqpoint{1.376cm}{1.588cm}}
\pgfpathlineto{\pgfqpoint{0cm}{1.588cm}}
\pgfpathclose
\pgfusepath{clip}
\begin{pgfscope}
\begin{pgfscope}
\pgfpathmoveto{\pgfqpoint{0cm}{0cm}}
\pgfpathlineto{\pgfqpoint{1.376cm}{0cm}}
\pgfpathlineto{\pgfqpoint{1.376cm}{1.588cm}}
\pgfpathlineto{\pgfqpoint{0cm}{1.588cm}}
\pgfpathclose
\pgfusepath{clip}
\begin{pgfscope}
\begin{pgfscope}
\definecolor{eps2pgf_color}{gray}{0.976471}\pgfsetstrokecolor{eps2pgf_color}\pgfsetfillcolor{eps2pgf_color}
\pgfpathmoveto{\pgfqpoint{0cm}{0cm}}
\pgfpathlineto{\pgfqpoint{1.376cm}{0cm}}
\pgfpathlineto{\pgfqpoint{1.376cm}{1.588cm}}
\pgfpathlineto{\pgfqpoint{0cm}{1.588cm}}
\pgfpathclose
\pgfusepath{fill}
\end{pgfscope}
\begin{pgfscope}
\pgfsetdash{}{0cm}
\pgfsetlinewidth{0.818mm}
\pgfsetroundcap
\pgfsetroundjoin
\pgfsetmiterlimit{7.0}
\definecolor{eps2pgf_color}{gray}{0}\pgfsetstrokecolor{eps2pgf_color}\pgfsetfillcolor{eps2pgf_color}
\pgfpathmoveto{\pgfqpoint{0.117cm}{1.476cm}}
\pgfpathlineto{\pgfqpoint{0.682cm}{0.726cm}}
\pgfpathlineto{\pgfqpoint{1.246cm}{1.476cm}}
\pgfusepath{stroke}
\end{pgfscope}
\definecolor{eps2pgf_color}{gray}{0}\pgfsetstrokecolor{eps2pgf_color}\pgfsetfillcolor{eps2pgf_color}
\pgfpathmoveto{\pgfqpoint{0.273cm}{1.451cm}}
\pgfpathcurveto{\pgfqpoint{0.273cm}{1.487cm}}{\pgfqpoint{0.259cm}{1.522cm}}{\pgfqpoint{0.233cm}{1.547cm}}
\pgfpathcurveto{\pgfqpoint{0.207cm}{1.573cm}}{\pgfqpoint{0.173cm}{1.588cm}}{\pgfqpoint{0.137cm}{1.588cm}}
\pgfpathcurveto{\pgfqpoint{0.1cm}{1.588cm}}{\pgfqpoint{0.066cm}{1.573cm}}{\pgfqpoint{0.04cm}{1.547cm}}
\pgfpathcurveto{\pgfqpoint{0.014cm}{1.522cm}}{\pgfqpoint{0cm}{1.487cm}}{\pgfqpoint{0cm}{1.451cm}}
\pgfpathcurveto{\pgfqpoint{0cm}{1.414cm}}{\pgfqpoint{0.014cm}{1.379cm}}{\pgfqpoint{0.04cm}{1.354cm}}
\pgfpathcurveto{\pgfqpoint{0.066cm}{1.328cm}}{\pgfqpoint{0.1cm}{1.314cm}}{\pgfqpoint{0.137cm}{1.314cm}}
\pgfpathcurveto{\pgfqpoint{0.173cm}{1.314cm}}{\pgfqpoint{0.207cm}{1.328cm}}{\pgfqpoint{0.233cm}{1.354cm}}
\pgfpathcurveto{\pgfqpoint{0.259cm}{1.379cm}}{\pgfqpoint{0.273cm}{1.414cm}}{\pgfqpoint{0.273cm}{1.451cm}}
\pgfusepath{fill}
\pgfpathmoveto{\pgfqpoint{1.345cm}{1.426cm}}
\pgfpathcurveto{\pgfqpoint{1.345cm}{1.463cm}}{\pgfqpoint{1.331cm}{1.497cm}}{\pgfqpoint{1.305cm}{1.523cm}}
\pgfpathcurveto{\pgfqpoint{1.28cm}{1.549cm}}{\pgfqpoint{1.245cm}{1.563cm}}{\pgfqpoint{1.209cm}{1.563cm}}
\pgfpathcurveto{\pgfqpoint{1.172cm}{1.563cm}}{\pgfqpoint{1.138cm}{1.549cm}}{\pgfqpoint{1.112cm}{1.523cm}}
\pgfpathcurveto{\pgfqpoint{1.087cm}{1.497cm}}{\pgfqpoint{1.072cm}{1.463cm}}{\pgfqpoint{1.072cm}{1.426cm}}
\pgfpathcurveto{\pgfqpoint{1.072cm}{1.39cm}}{\pgfqpoint{1.087cm}{1.355cm}}{\pgfqpoint{1.112cm}{1.329cm}}
\pgfpathcurveto{\pgfqpoint{1.138cm}{1.304cm}}{\pgfqpoint{1.172cm}{1.289cm}}{\pgfqpoint{1.209cm}{1.289cm}}
\pgfpathcurveto{\pgfqpoint{1.245cm}{1.289cm}}{\pgfqpoint{1.28cm}{1.304cm}}{\pgfqpoint{1.305cm}{1.329cm}}
\pgfpathcurveto{\pgfqpoint{1.331cm}{1.355cm}}{\pgfqpoint{1.345cm}{1.39cm}}{\pgfqpoint{1.345cm}{1.426cm}}
\pgfusepath{fill}
\begin{pgfscope}
\pgfsetdash{}{0cm}
\pgfsetlinewidth{0.818mm}
\pgfsetroundcap
\pgfsetmiterlimit{4.0}
\pgfpathmoveto{\pgfqpoint{0.682cm}{0.726cm}}
\pgfpathlineto{\pgfqpoint{0.682cm}{0.097cm}}
\pgfusepath{stroke}
\end{pgfscope}
\end{pgfscope}
\end{pgfscope}
\end{pgfscope}
\end{tikzpicture}}} \succ \nabla_{\varepsilon} \phi_{M, \varepsilon}\\
    & = & U_{M, \varepsilon} + 3\lambda X_{M, \varepsilon}^{\!\resizebox{0.6em}{!}{
\begin{tikzpicture}
\pgfpathmoveto{\pgfqpoint{0cm}{0cm}}
\pgfpathlineto{\pgfqpoint{1.376cm}{0cm}}
\pgfpathlineto{\pgfqpoint{1.376cm}{1.588cm}}
\pgfpathlineto{\pgfqpoint{0cm}{1.588cm}}
\pgfpathclose
\pgfusepath{clip}
\begin{pgfscope}
\begin{pgfscope}
\pgfpathmoveto{\pgfqpoint{0cm}{0cm}}
\pgfpathlineto{\pgfqpoint{1.376cm}{0cm}}
\pgfpathlineto{\pgfqpoint{1.376cm}{1.588cm}}
\pgfpathlineto{\pgfqpoint{0cm}{1.588cm}}
\pgfpathclose
\pgfusepath{clip}
\begin{pgfscope}
\begin{pgfscope}
\definecolor{eps2pgf_color}{gray}{0.976471}\pgfsetstrokecolor{eps2pgf_color}\pgfsetfillcolor{eps2pgf_color}
\pgfpathmoveto{\pgfqpoint{0cm}{0cm}}
\pgfpathlineto{\pgfqpoint{1.376cm}{0cm}}
\pgfpathlineto{\pgfqpoint{1.376cm}{1.588cm}}
\pgfpathlineto{\pgfqpoint{0cm}{1.588cm}}
\pgfpathclose
\pgfusepath{fill}
\end{pgfscope}
\begin{pgfscope}
\pgfsetdash{}{0cm}
\pgfsetlinewidth{0.818mm}
\pgfsetroundcap
\pgfsetroundjoin
\pgfsetmiterlimit{7.0}
\definecolor{eps2pgf_color}{gray}{0}\pgfsetstrokecolor{eps2pgf_color}\pgfsetfillcolor{eps2pgf_color}
\pgfpathmoveto{\pgfqpoint{0.117cm}{1.476cm}}
\pgfpathlineto{\pgfqpoint{0.682cm}{0.726cm}}
\pgfpathlineto{\pgfqpoint{1.246cm}{1.476cm}}
\pgfusepath{stroke}
\end{pgfscope}
\definecolor{eps2pgf_color}{gray}{0}\pgfsetstrokecolor{eps2pgf_color}\pgfsetfillcolor{eps2pgf_color}
\pgfpathmoveto{\pgfqpoint{0.273cm}{1.451cm}}
\pgfpathcurveto{\pgfqpoint{0.273cm}{1.487cm}}{\pgfqpoint{0.259cm}{1.522cm}}{\pgfqpoint{0.233cm}{1.547cm}}
\pgfpathcurveto{\pgfqpoint{0.207cm}{1.573cm}}{\pgfqpoint{0.173cm}{1.588cm}}{\pgfqpoint{0.137cm}{1.588cm}}
\pgfpathcurveto{\pgfqpoint{0.1cm}{1.588cm}}{\pgfqpoint{0.066cm}{1.573cm}}{\pgfqpoint{0.04cm}{1.547cm}}
\pgfpathcurveto{\pgfqpoint{0.014cm}{1.522cm}}{\pgfqpoint{0cm}{1.487cm}}{\pgfqpoint{0cm}{1.451cm}}
\pgfpathcurveto{\pgfqpoint{0cm}{1.414cm}}{\pgfqpoint{0.014cm}{1.379cm}}{\pgfqpoint{0.04cm}{1.354cm}}
\pgfpathcurveto{\pgfqpoint{0.066cm}{1.328cm}}{\pgfqpoint{0.1cm}{1.314cm}}{\pgfqpoint{0.137cm}{1.314cm}}
\pgfpathcurveto{\pgfqpoint{0.173cm}{1.314cm}}{\pgfqpoint{0.207cm}{1.328cm}}{\pgfqpoint{0.233cm}{1.354cm}}
\pgfpathcurveto{\pgfqpoint{0.259cm}{1.379cm}}{\pgfqpoint{0.273cm}{1.414cm}}{\pgfqpoint{0.273cm}{1.451cm}}
\pgfusepath{fill}
\pgfpathmoveto{\pgfqpoint{1.345cm}{1.426cm}}
\pgfpathcurveto{\pgfqpoint{1.345cm}{1.463cm}}{\pgfqpoint{1.331cm}{1.497cm}}{\pgfqpoint{1.305cm}{1.523cm}}
\pgfpathcurveto{\pgfqpoint{1.28cm}{1.549cm}}{\pgfqpoint{1.245cm}{1.563cm}}{\pgfqpoint{1.209cm}{1.563cm}}
\pgfpathcurveto{\pgfqpoint{1.172cm}{1.563cm}}{\pgfqpoint{1.138cm}{1.549cm}}{\pgfqpoint{1.112cm}{1.523cm}}
\pgfpathcurveto{\pgfqpoint{1.087cm}{1.497cm}}{\pgfqpoint{1.072cm}{1.463cm}}{\pgfqpoint{1.072cm}{1.426cm}}
\pgfpathcurveto{\pgfqpoint{1.072cm}{1.39cm}}{\pgfqpoint{1.087cm}{1.355cm}}{\pgfqpoint{1.112cm}{1.329cm}}
\pgfpathcurveto{\pgfqpoint{1.138cm}{1.304cm}}{\pgfqpoint{1.172cm}{1.289cm}}{\pgfqpoint{1.209cm}{1.289cm}}
\pgfpathcurveto{\pgfqpoint{1.245cm}{1.289cm}}{\pgfqpoint{1.28cm}{1.304cm}}{\pgfqpoint{1.305cm}{1.329cm}}
\pgfpathcurveto{\pgfqpoint{1.331cm}{1.355cm}}{\pgfqpoint{1.345cm}{1.39cm}}{\pgfqpoint{1.345cm}{1.426cm}}
\pgfusepath{fill}
\begin{pgfscope}
\pgfsetdash{}{0cm}
\pgfsetlinewidth{0.818mm}
\pgfsetroundcap
\pgfsetmiterlimit{4.0}
\pgfpathmoveto{\pgfqpoint{0.682cm}{0.726cm}}
\pgfpathlineto{\pgfqpoint{0.682cm}{0.097cm}}
\pgfusepath{stroke}
\end{pgfscope}
\end{pgfscope}
\end{pgfscope}
\end{pgfscope}
\end{tikzpicture}}} \succ
    \LL_{\varepsilon} \phi - 6\lambda \nabla_{\varepsilon} X_{M,
    \varepsilon}^{\!\resizebox{0.6em}{!}{
\begin{tikzpicture}
\pgfpathmoveto{\pgfqpoint{0cm}{0cm}}
\pgfpathlineto{\pgfqpoint{1.376cm}{0cm}}
\pgfpathlineto{\pgfqpoint{1.376cm}{1.588cm}}
\pgfpathlineto{\pgfqpoint{0cm}{1.588cm}}
\pgfpathclose
\pgfusepath{clip}
\begin{pgfscope}
\begin{pgfscope}
\pgfpathmoveto{\pgfqpoint{0cm}{0cm}}
\pgfpathlineto{\pgfqpoint{1.376cm}{0cm}}
\pgfpathlineto{\pgfqpoint{1.376cm}{1.588cm}}
\pgfpathlineto{\pgfqpoint{0cm}{1.588cm}}
\pgfpathclose
\pgfusepath{clip}
\begin{pgfscope}
\begin{pgfscope}
\definecolor{eps2pgf_color}{gray}{0.976471}\pgfsetstrokecolor{eps2pgf_color}\pgfsetfillcolor{eps2pgf_color}
\pgfpathmoveto{\pgfqpoint{0cm}{0cm}}
\pgfpathlineto{\pgfqpoint{1.376cm}{0cm}}
\pgfpathlineto{\pgfqpoint{1.376cm}{1.588cm}}
\pgfpathlineto{\pgfqpoint{0cm}{1.588cm}}
\pgfpathclose
\pgfusepath{fill}
\end{pgfscope}
\begin{pgfscope}
\pgfsetdash{}{0cm}
\pgfsetlinewidth{0.818mm}
\pgfsetroundcap
\pgfsetroundjoin
\pgfsetmiterlimit{7.0}
\definecolor{eps2pgf_color}{gray}{0}\pgfsetstrokecolor{eps2pgf_color}\pgfsetfillcolor{eps2pgf_color}
\pgfpathmoveto{\pgfqpoint{0.117cm}{1.476cm}}
\pgfpathlineto{\pgfqpoint{0.682cm}{0.726cm}}
\pgfpathlineto{\pgfqpoint{1.246cm}{1.476cm}}
\pgfusepath{stroke}
\end{pgfscope}
\definecolor{eps2pgf_color}{gray}{0}\pgfsetstrokecolor{eps2pgf_color}\pgfsetfillcolor{eps2pgf_color}
\pgfpathmoveto{\pgfqpoint{0.273cm}{1.451cm}}
\pgfpathcurveto{\pgfqpoint{0.273cm}{1.487cm}}{\pgfqpoint{0.259cm}{1.522cm}}{\pgfqpoint{0.233cm}{1.547cm}}
\pgfpathcurveto{\pgfqpoint{0.207cm}{1.573cm}}{\pgfqpoint{0.173cm}{1.588cm}}{\pgfqpoint{0.137cm}{1.588cm}}
\pgfpathcurveto{\pgfqpoint{0.1cm}{1.588cm}}{\pgfqpoint{0.066cm}{1.573cm}}{\pgfqpoint{0.04cm}{1.547cm}}
\pgfpathcurveto{\pgfqpoint{0.014cm}{1.522cm}}{\pgfqpoint{0cm}{1.487cm}}{\pgfqpoint{0cm}{1.451cm}}
\pgfpathcurveto{\pgfqpoint{0cm}{1.414cm}}{\pgfqpoint{0.014cm}{1.379cm}}{\pgfqpoint{0.04cm}{1.354cm}}
\pgfpathcurveto{\pgfqpoint{0.066cm}{1.328cm}}{\pgfqpoint{0.1cm}{1.314cm}}{\pgfqpoint{0.137cm}{1.314cm}}
\pgfpathcurveto{\pgfqpoint{0.173cm}{1.314cm}}{\pgfqpoint{0.207cm}{1.328cm}}{\pgfqpoint{0.233cm}{1.354cm}}
\pgfpathcurveto{\pgfqpoint{0.259cm}{1.379cm}}{\pgfqpoint{0.273cm}{1.414cm}}{\pgfqpoint{0.273cm}{1.451cm}}
\pgfusepath{fill}
\pgfpathmoveto{\pgfqpoint{1.345cm}{1.426cm}}
\pgfpathcurveto{\pgfqpoint{1.345cm}{1.463cm}}{\pgfqpoint{1.331cm}{1.497cm}}{\pgfqpoint{1.305cm}{1.523cm}}
\pgfpathcurveto{\pgfqpoint{1.28cm}{1.549cm}}{\pgfqpoint{1.245cm}{1.563cm}}{\pgfqpoint{1.209cm}{1.563cm}}
\pgfpathcurveto{\pgfqpoint{1.172cm}{1.563cm}}{\pgfqpoint{1.138cm}{1.549cm}}{\pgfqpoint{1.112cm}{1.523cm}}
\pgfpathcurveto{\pgfqpoint{1.087cm}{1.497cm}}{\pgfqpoint{1.072cm}{1.463cm}}{\pgfqpoint{1.072cm}{1.426cm}}
\pgfpathcurveto{\pgfqpoint{1.072cm}{1.39cm}}{\pgfqpoint{1.087cm}{1.355cm}}{\pgfqpoint{1.112cm}{1.329cm}}
\pgfpathcurveto{\pgfqpoint{1.138cm}{1.304cm}}{\pgfqpoint{1.172cm}{1.289cm}}{\pgfqpoint{1.209cm}{1.289cm}}
\pgfpathcurveto{\pgfqpoint{1.245cm}{1.289cm}}{\pgfqpoint{1.28cm}{1.304cm}}{\pgfqpoint{1.305cm}{1.329cm}}
\pgfpathcurveto{\pgfqpoint{1.331cm}{1.355cm}}{\pgfqpoint{1.345cm}{1.39cm}}{\pgfqpoint{1.345cm}{1.426cm}}
\pgfusepath{fill}
\begin{pgfscope}
\pgfsetdash{}{0cm}
\pgfsetlinewidth{0.818mm}
\pgfsetroundcap
\pgfsetmiterlimit{4.0}
\pgfpathmoveto{\pgfqpoint{0.682cm}{0.726cm}}
\pgfpathlineto{\pgfqpoint{0.682cm}{0.097cm}}
\pgfusepath{stroke}
\end{pgfscope}
\end{pgfscope}
\end{pgfscope}
\end{pgfscope}
\end{tikzpicture}}} \succ \nabla_{\varepsilon} \phi_{M, \varepsilon}\\
    & = & U_{M, \varepsilon} + 3\lambda X_{M, \varepsilon}^{\!\resizebox{0.6em}{!}{
\begin{tikzpicture}
\pgfpathmoveto{\pgfqpoint{0cm}{0cm}}
\pgfpathlineto{\pgfqpoint{1.376cm}{0cm}}
\pgfpathlineto{\pgfqpoint{1.376cm}{1.588cm}}
\pgfpathlineto{\pgfqpoint{0cm}{1.588cm}}
\pgfpathclose
\pgfusepath{clip}
\begin{pgfscope}
\begin{pgfscope}
\pgfpathmoveto{\pgfqpoint{0cm}{0cm}}
\pgfpathlineto{\pgfqpoint{1.376cm}{0cm}}
\pgfpathlineto{\pgfqpoint{1.376cm}{1.588cm}}
\pgfpathlineto{\pgfqpoint{0cm}{1.588cm}}
\pgfpathclose
\pgfusepath{clip}
\begin{pgfscope}
\begin{pgfscope}
\definecolor{eps2pgf_color}{gray}{0.976471}\pgfsetstrokecolor{eps2pgf_color}\pgfsetfillcolor{eps2pgf_color}
\pgfpathmoveto{\pgfqpoint{0cm}{0cm}}
\pgfpathlineto{\pgfqpoint{1.376cm}{0cm}}
\pgfpathlineto{\pgfqpoint{1.376cm}{1.588cm}}
\pgfpathlineto{\pgfqpoint{0cm}{1.588cm}}
\pgfpathclose
\pgfusepath{fill}
\end{pgfscope}
\begin{pgfscope}
\pgfsetdash{}{0cm}
\pgfsetlinewidth{0.818mm}
\pgfsetroundcap
\pgfsetroundjoin
\pgfsetmiterlimit{7.0}
\definecolor{eps2pgf_color}{gray}{0}\pgfsetstrokecolor{eps2pgf_color}\pgfsetfillcolor{eps2pgf_color}
\pgfpathmoveto{\pgfqpoint{0.117cm}{1.476cm}}
\pgfpathlineto{\pgfqpoint{0.682cm}{0.726cm}}
\pgfpathlineto{\pgfqpoint{1.246cm}{1.476cm}}
\pgfusepath{stroke}
\end{pgfscope}
\definecolor{eps2pgf_color}{gray}{0}\pgfsetstrokecolor{eps2pgf_color}\pgfsetfillcolor{eps2pgf_color}
\pgfpathmoveto{\pgfqpoint{0.273cm}{1.451cm}}
\pgfpathcurveto{\pgfqpoint{0.273cm}{1.487cm}}{\pgfqpoint{0.259cm}{1.522cm}}{\pgfqpoint{0.233cm}{1.547cm}}
\pgfpathcurveto{\pgfqpoint{0.207cm}{1.573cm}}{\pgfqpoint{0.173cm}{1.588cm}}{\pgfqpoint{0.137cm}{1.588cm}}
\pgfpathcurveto{\pgfqpoint{0.1cm}{1.588cm}}{\pgfqpoint{0.066cm}{1.573cm}}{\pgfqpoint{0.04cm}{1.547cm}}
\pgfpathcurveto{\pgfqpoint{0.014cm}{1.522cm}}{\pgfqpoint{0cm}{1.487cm}}{\pgfqpoint{0cm}{1.451cm}}
\pgfpathcurveto{\pgfqpoint{0cm}{1.414cm}}{\pgfqpoint{0.014cm}{1.379cm}}{\pgfqpoint{0.04cm}{1.354cm}}
\pgfpathcurveto{\pgfqpoint{0.066cm}{1.328cm}}{\pgfqpoint{0.1cm}{1.314cm}}{\pgfqpoint{0.137cm}{1.314cm}}
\pgfpathcurveto{\pgfqpoint{0.173cm}{1.314cm}}{\pgfqpoint{0.207cm}{1.328cm}}{\pgfqpoint{0.233cm}{1.354cm}}
\pgfpathcurveto{\pgfqpoint{0.259cm}{1.379cm}}{\pgfqpoint{0.273cm}{1.414cm}}{\pgfqpoint{0.273cm}{1.451cm}}
\pgfusepath{fill}
\pgfpathmoveto{\pgfqpoint{1.345cm}{1.426cm}}
\pgfpathcurveto{\pgfqpoint{1.345cm}{1.463cm}}{\pgfqpoint{1.331cm}{1.497cm}}{\pgfqpoint{1.305cm}{1.523cm}}
\pgfpathcurveto{\pgfqpoint{1.28cm}{1.549cm}}{\pgfqpoint{1.245cm}{1.563cm}}{\pgfqpoint{1.209cm}{1.563cm}}
\pgfpathcurveto{\pgfqpoint{1.172cm}{1.563cm}}{\pgfqpoint{1.138cm}{1.549cm}}{\pgfqpoint{1.112cm}{1.523cm}}
\pgfpathcurveto{\pgfqpoint{1.087cm}{1.497cm}}{\pgfqpoint{1.072cm}{1.463cm}}{\pgfqpoint{1.072cm}{1.426cm}}
\pgfpathcurveto{\pgfqpoint{1.072cm}{1.39cm}}{\pgfqpoint{1.087cm}{1.355cm}}{\pgfqpoint{1.112cm}{1.329cm}}
\pgfpathcurveto{\pgfqpoint{1.138cm}{1.304cm}}{\pgfqpoint{1.172cm}{1.289cm}}{\pgfqpoint{1.209cm}{1.289cm}}
\pgfpathcurveto{\pgfqpoint{1.245cm}{1.289cm}}{\pgfqpoint{1.28cm}{1.304cm}}{\pgfqpoint{1.305cm}{1.329cm}}
\pgfpathcurveto{\pgfqpoint{1.331cm}{1.355cm}}{\pgfqpoint{1.345cm}{1.39cm}}{\pgfqpoint{1.345cm}{1.426cm}}
\pgfusepath{fill}
\begin{pgfscope}
\pgfsetdash{}{0cm}
\pgfsetlinewidth{0.818mm}
\pgfsetroundcap
\pgfsetmiterlimit{4.0}
\pgfpathmoveto{\pgfqpoint{0.682cm}{0.726cm}}
\pgfpathlineto{\pgfqpoint{0.682cm}{0.097cm}}
\pgfusepath{stroke}
\end{pgfscope}
\end{pgfscope}
\end{pgfscope}
\end{pgfscope}
\end{tikzpicture}}} \succ (- 3\lambda
    \llbracket X_{M, \varepsilon}^2 \rrbracket \succ \phi_{M, \varepsilon} +
    U_{M, \varepsilon}) - 6\lambda \nabla_{\varepsilon} X_{M, \varepsilon}^{\!\resizebox{0.6em}{!}{
\begin{tikzpicture}
\pgfpathmoveto{\pgfqpoint{0cm}{0cm}}
\pgfpathlineto{\pgfqpoint{1.376cm}{0cm}}
\pgfpathlineto{\pgfqpoint{1.376cm}{1.588cm}}
\pgfpathlineto{\pgfqpoint{0cm}{1.588cm}}
\pgfpathclose
\pgfusepath{clip}
\begin{pgfscope}
\begin{pgfscope}
\pgfpathmoveto{\pgfqpoint{0cm}{0cm}}
\pgfpathlineto{\pgfqpoint{1.376cm}{0cm}}
\pgfpathlineto{\pgfqpoint{1.376cm}{1.588cm}}
\pgfpathlineto{\pgfqpoint{0cm}{1.588cm}}
\pgfpathclose
\pgfusepath{clip}
\begin{pgfscope}
\begin{pgfscope}
\definecolor{eps2pgf_color}{gray}{0.976471}\pgfsetstrokecolor{eps2pgf_color}\pgfsetfillcolor{eps2pgf_color}
\pgfpathmoveto{\pgfqpoint{0cm}{0cm}}
\pgfpathlineto{\pgfqpoint{1.376cm}{0cm}}
\pgfpathlineto{\pgfqpoint{1.376cm}{1.588cm}}
\pgfpathlineto{\pgfqpoint{0cm}{1.588cm}}
\pgfpathclose
\pgfusepath{fill}
\end{pgfscope}
\begin{pgfscope}
\pgfsetdash{}{0cm}
\pgfsetlinewidth{0.818mm}
\pgfsetroundcap
\pgfsetroundjoin
\pgfsetmiterlimit{7.0}
\definecolor{eps2pgf_color}{gray}{0}\pgfsetstrokecolor{eps2pgf_color}\pgfsetfillcolor{eps2pgf_color}
\pgfpathmoveto{\pgfqpoint{0.117cm}{1.476cm}}
\pgfpathlineto{\pgfqpoint{0.682cm}{0.726cm}}
\pgfpathlineto{\pgfqpoint{1.246cm}{1.476cm}}
\pgfusepath{stroke}
\end{pgfscope}
\definecolor{eps2pgf_color}{gray}{0}\pgfsetstrokecolor{eps2pgf_color}\pgfsetfillcolor{eps2pgf_color}
\pgfpathmoveto{\pgfqpoint{0.273cm}{1.451cm}}
\pgfpathcurveto{\pgfqpoint{0.273cm}{1.487cm}}{\pgfqpoint{0.259cm}{1.522cm}}{\pgfqpoint{0.233cm}{1.547cm}}
\pgfpathcurveto{\pgfqpoint{0.207cm}{1.573cm}}{\pgfqpoint{0.173cm}{1.588cm}}{\pgfqpoint{0.137cm}{1.588cm}}
\pgfpathcurveto{\pgfqpoint{0.1cm}{1.588cm}}{\pgfqpoint{0.066cm}{1.573cm}}{\pgfqpoint{0.04cm}{1.547cm}}
\pgfpathcurveto{\pgfqpoint{0.014cm}{1.522cm}}{\pgfqpoint{0cm}{1.487cm}}{\pgfqpoint{0cm}{1.451cm}}
\pgfpathcurveto{\pgfqpoint{0cm}{1.414cm}}{\pgfqpoint{0.014cm}{1.379cm}}{\pgfqpoint{0.04cm}{1.354cm}}
\pgfpathcurveto{\pgfqpoint{0.066cm}{1.328cm}}{\pgfqpoint{0.1cm}{1.314cm}}{\pgfqpoint{0.137cm}{1.314cm}}
\pgfpathcurveto{\pgfqpoint{0.173cm}{1.314cm}}{\pgfqpoint{0.207cm}{1.328cm}}{\pgfqpoint{0.233cm}{1.354cm}}
\pgfpathcurveto{\pgfqpoint{0.259cm}{1.379cm}}{\pgfqpoint{0.273cm}{1.414cm}}{\pgfqpoint{0.273cm}{1.451cm}}
\pgfusepath{fill}
\pgfpathmoveto{\pgfqpoint{1.345cm}{1.426cm}}
\pgfpathcurveto{\pgfqpoint{1.345cm}{1.463cm}}{\pgfqpoint{1.331cm}{1.497cm}}{\pgfqpoint{1.305cm}{1.523cm}}
\pgfpathcurveto{\pgfqpoint{1.28cm}{1.549cm}}{\pgfqpoint{1.245cm}{1.563cm}}{\pgfqpoint{1.209cm}{1.563cm}}
\pgfpathcurveto{\pgfqpoint{1.172cm}{1.563cm}}{\pgfqpoint{1.138cm}{1.549cm}}{\pgfqpoint{1.112cm}{1.523cm}}
\pgfpathcurveto{\pgfqpoint{1.087cm}{1.497cm}}{\pgfqpoint{1.072cm}{1.463cm}}{\pgfqpoint{1.072cm}{1.426cm}}
\pgfpathcurveto{\pgfqpoint{1.072cm}{1.39cm}}{\pgfqpoint{1.087cm}{1.355cm}}{\pgfqpoint{1.112cm}{1.329cm}}
\pgfpathcurveto{\pgfqpoint{1.138cm}{1.304cm}}{\pgfqpoint{1.172cm}{1.289cm}}{\pgfqpoint{1.209cm}{1.289cm}}
\pgfpathcurveto{\pgfqpoint{1.245cm}{1.289cm}}{\pgfqpoint{1.28cm}{1.304cm}}{\pgfqpoint{1.305cm}{1.329cm}}
\pgfpathcurveto{\pgfqpoint{1.331cm}{1.355cm}}{\pgfqpoint{1.345cm}{1.39cm}}{\pgfqpoint{1.345cm}{1.426cm}}
\pgfusepath{fill}
\begin{pgfscope}
\pgfsetdash{}{0cm}
\pgfsetlinewidth{0.818mm}
\pgfsetroundcap
\pgfsetmiterlimit{4.0}
\pgfpathmoveto{\pgfqpoint{0.682cm}{0.726cm}}
\pgfpathlineto{\pgfqpoint{0.682cm}{0.097cm}}
\pgfusepath{stroke}
\end{pgfscope}
\end{pgfscope}
\end{pgfscope}
\end{pgfscope}
\end{tikzpicture}}}
    \succ \nabla_{\varepsilon} \phi_{M, \varepsilon},
  \end{array} \label{eq:chi11}
\end{equation}
where the bilinear form $\nabla_{\varepsilon} f \prec \nabla_{\varepsilon} g$
is defined by
\[ \nabla_{\varepsilon} f \prec \nabla_{\varepsilon} g \assign \frac{1}{2}
   (\Delta_{\varepsilon} (f \prec g) - \Delta_{\varepsilon} f \prec g - f
   \prec \Delta_{\varepsilon} g) \]
and can be controlled as in the proof of Lemma~\ref{lem:comm1}.

Next, we state a regularity result for $\chi_{M,\varepsilon}$, proof of which is postponed to Appendix \ref{s:chi-reg}. While it is in principle possible to keep track of the exact dependence of the bounds on $\lambda$ we do not pursue it any further since there seems to be no interesting application of such bounds. Nevertheless, it can be checked that the bounds in this section remain uniform over $\lambda$ belonging to any bounded subset of $[0,\infty)$.

\begin{proposition}
  \label{prop:reg}Let $\rho$ be a weight such that $\rho^{\iota} \in L^{4, 0}$
  for some $\iota \in (0, 1)$. Let $\phi_{M, \varepsilon}$ be a solution to
  {\eqref{eq:phiU}} and let $\chi_{M, \varepsilon}$ be given by
  {\eqref{eq:chi1}}. Then
  \[ \| \rho^4 \chi_{M, \varepsilon} \|_{L^1_T B_{1, 1}^{1 + 3 \kappa,
     \varepsilon}} \leqslant C_{T,m^2,\lambda} Q_{\rho} (\mathbb{X}_{M,\varepsilon}) (1+\| \rho^2 \phi_{
     M,\varepsilon} (0)\|_{L^{2, \varepsilon}}).\]
\end{proposition}

We apply this result in order to
deduce tightness of the sequence $(\varphi_{M, \varepsilon})_{M, \varepsilon}$
as time-dependent stochastic processes. In other words, in contrast to
Theorem~\ref{thm:tight}, where we only proved tightness for a fixed time $t
\geqslant 0$, it is necessary to establish uniform time regularity of
$(\varphi_{M, \varepsilon})_{M, \varepsilon}$. To this end, we recall the
decompositions \[\varphi_{M, \varepsilon} = X_{M, \varepsilon} + Y_{M,
\varepsilon} + \phi_{M, \varepsilon} =
X_{M, \varepsilon} - \lambda X^{\!\resizebox{0.6em}{!}{
\begin{tikzpicture}
\pgfpathmoveto{\pgfqpoint{0cm}{-0.035cm}}
\pgfpathlineto{\pgfqpoint{1.376cm}{-0.035cm}}
\pgfpathlineto{\pgfqpoint{1.376cm}{1.552cm}}
\pgfpathlineto{\pgfqpoint{0cm}{1.552cm}}
\pgfpathclose
\pgfusepath{clip}
\begin{pgfscope}
\begin{pgfscope}
\pgfpathmoveto{\pgfqpoint{0cm}{-0.035cm}}
\pgfpathlineto{\pgfqpoint{1.376cm}{-0.035cm}}
\pgfpathlineto{\pgfqpoint{1.376cm}{1.552cm}}
\pgfpathlineto{\pgfqpoint{0cm}{1.552cm}}
\pgfpathclose
\pgfusepath{clip}
\begin{pgfscope}
\begin{pgfscope}
\pgfsetdash{}{0cm}
\pgfsetlinewidth{0.818mm}
\pgfsetroundcap
\pgfsetroundjoin
\pgfsetmiterlimit{7.0}
\definecolor{eps2pgf_color}{gray}{0}\pgfsetstrokecolor{eps2pgf_color}\pgfsetfillcolor{eps2pgf_color}
\pgfpathmoveto{\pgfqpoint{0.117cm}{1.421cm}}
\pgfpathlineto{\pgfqpoint{0.682cm}{0.671cm}}
\pgfpathlineto{\pgfqpoint{1.246cm}{1.421cm}}
\pgfusepath{stroke}
\end{pgfscope}
\definecolor{eps2pgf_color}{gray}{0}\pgfsetstrokecolor{eps2pgf_color}\pgfsetfillcolor{eps2pgf_color}
\pgfpathmoveto{\pgfqpoint{0.273cm}{1.395cm}}
\pgfpathcurveto{\pgfqpoint{0.273cm}{1.432cm}}{\pgfqpoint{0.259cm}{1.467cm}}{\pgfqpoint{0.233cm}{1.492cm}}
\pgfpathcurveto{\pgfqpoint{0.207cm}{1.518cm}}{\pgfqpoint{0.173cm}{1.532cm}}{\pgfqpoint{0.137cm}{1.532cm}}
\pgfpathcurveto{\pgfqpoint{0.1cm}{1.532cm}}{\pgfqpoint{0.066cm}{1.518cm}}{\pgfqpoint{0.04cm}{1.492cm}}
\pgfpathcurveto{\pgfqpoint{0.014cm}{1.467cm}}{\pgfqpoint{0cm}{1.432cm}}{\pgfqpoint{0cm}{1.395cm}}
\pgfpathcurveto{\pgfqpoint{0cm}{1.359cm}}{\pgfqpoint{0.014cm}{1.324cm}}{\pgfqpoint{0.04cm}{1.299cm}}
\pgfpathcurveto{\pgfqpoint{0.066cm}{1.273cm}}{\pgfqpoint{0.1cm}{1.258cm}}{\pgfqpoint{0.137cm}{1.258cm}}
\pgfpathcurveto{\pgfqpoint{0.173cm}{1.258cm}}{\pgfqpoint{0.207cm}{1.273cm}}{\pgfqpoint{0.233cm}{1.299cm}}
\pgfpathcurveto{\pgfqpoint{0.259cm}{1.324cm}}{\pgfqpoint{0.273cm}{1.359cm}}{\pgfqpoint{0.273cm}{1.395cm}}
\pgfusepath{fill}
\begin{pgfscope}
\pgfsetdash{}{0cm}
\pgfsetlinewidth{0.818mm}
\pgfsetmiterlimit{7.0}
\pgfpathmoveto{\pgfqpoint{0.682cm}{0.671cm}}
\pgfpathlineto{\pgfqpoint{0.679cm}{1.418cm}}
\pgfusepath{stroke}
\end{pgfscope}
\pgfpathmoveto{\pgfqpoint{0.815cm}{1.399cm}}
\pgfpathcurveto{\pgfqpoint{0.815cm}{1.435cm}}{\pgfqpoint{0.801cm}{1.47cm}}{\pgfqpoint{0.775cm}{1.496cm}}
\pgfpathcurveto{\pgfqpoint{0.75cm}{1.521cm}}{\pgfqpoint{0.715cm}{1.536cm}}{\pgfqpoint{0.679cm}{1.536cm}}
\pgfpathcurveto{\pgfqpoint{0.643cm}{1.536cm}}{\pgfqpoint{0.608cm}{1.521cm}}{\pgfqpoint{0.582cm}{1.496cm}}
\pgfpathcurveto{\pgfqpoint{0.557cm}{1.47cm}}{\pgfqpoint{0.542cm}{1.435cm}}{\pgfqpoint{0.542cm}{1.399cm}}
\pgfpathcurveto{\pgfqpoint{0.542cm}{1.363cm}}{\pgfqpoint{0.557cm}{1.328cm}}{\pgfqpoint{0.582cm}{1.302cm}}
\pgfpathcurveto{\pgfqpoint{0.608cm}{1.276cm}}{\pgfqpoint{0.643cm}{1.262cm}}{\pgfqpoint{0.679cm}{1.262cm}}
\pgfpathcurveto{\pgfqpoint{0.715cm}{1.262cm}}{\pgfqpoint{0.75cm}{1.276cm}}{\pgfqpoint{0.775cm}{1.302cm}}
\pgfpathcurveto{\pgfqpoint{0.801cm}{1.328cm}}{\pgfqpoint{0.815cm}{1.363cm}}{\pgfqpoint{0.815cm}{1.399cm}}
\pgfusepath{fill}
\pgfpathmoveto{\pgfqpoint{1.345cm}{1.371cm}}
\pgfpathcurveto{\pgfqpoint{1.345cm}{1.408cm}}{\pgfqpoint{1.331cm}{1.442cm}}{\pgfqpoint{1.305cm}{1.468cm}}
\pgfpathcurveto{\pgfqpoint{1.28cm}{1.494cm}}{\pgfqpoint{1.245cm}{1.508cm}}{\pgfqpoint{1.209cm}{1.508cm}}
\pgfpathcurveto{\pgfqpoint{1.172cm}{1.508cm}}{\pgfqpoint{1.138cm}{1.494cm}}{\pgfqpoint{1.112cm}{1.468cm}}
\pgfpathcurveto{\pgfqpoint{1.087cm}{1.442cm}}{\pgfqpoint{1.072cm}{1.408cm}}{\pgfqpoint{1.072cm}{1.371cm}}
\pgfpathcurveto{\pgfqpoint{1.072cm}{1.335cm}}{\pgfqpoint{1.087cm}{1.3cm}}{\pgfqpoint{1.112cm}{1.274cm}}
\pgfpathcurveto{\pgfqpoint{1.138cm}{1.249cm}}{\pgfqpoint{1.172cm}{1.234cm}}{\pgfqpoint{1.209cm}{1.234cm}}
\pgfpathcurveto{\pgfqpoint{1.245cm}{1.234cm}}{\pgfqpoint{1.28cm}{1.249cm}}{\pgfqpoint{1.305cm}{1.274cm}}
\pgfpathcurveto{\pgfqpoint{1.331cm}{1.3cm}}{\pgfqpoint{1.345cm}{1.335cm}}{\pgfqpoint{1.345cm}{1.371cm}}
\pgfusepath{fill}
\begin{pgfscope}
\pgfsetdash{}{0cm}
\pgfsetlinewidth{0.818mm}
\pgfsetroundcap
\pgfsetmiterlimit{4.0}
\pgfpathmoveto{\pgfqpoint{0.682cm}{0.671cm}}
\pgfpathlineto{\pgfqpoint{0.682cm}{0.042cm}}
\pgfusepath{stroke}
\end{pgfscope}
\end{pgfscope}
\end{pgfscope}
\end{pgfscope}
\end{tikzpicture}}}_{M, \varepsilon} + \zeta_{M, \varepsilon}\]
with
\begin{equation}
  \zeta_{M, \varepsilon} = Y_{M, \varepsilon} + \lambda X_{M, \varepsilon}^{\!\resizebox{0.6em}{!}{
\begin{tikzpicture}
\pgfpathmoveto{\pgfqpoint{0cm}{-0.035cm}}
\pgfpathlineto{\pgfqpoint{1.376cm}{-0.035cm}}
\pgfpathlineto{\pgfqpoint{1.376cm}{1.552cm}}
\pgfpathlineto{\pgfqpoint{0cm}{1.552cm}}
\pgfpathclose
\pgfusepath{clip}
\begin{pgfscope}
\begin{pgfscope}
\pgfpathmoveto{\pgfqpoint{0cm}{-0.035cm}}
\pgfpathlineto{\pgfqpoint{1.376cm}{-0.035cm}}
\pgfpathlineto{\pgfqpoint{1.376cm}{1.552cm}}
\pgfpathlineto{\pgfqpoint{0cm}{1.552cm}}
\pgfpathclose
\pgfusepath{clip}
\begin{pgfscope}
\begin{pgfscope}
\pgfsetdash{}{0cm}
\pgfsetlinewidth{0.818mm}
\pgfsetroundcap
\pgfsetroundjoin
\pgfsetmiterlimit{7.0}
\definecolor{eps2pgf_color}{gray}{0}\pgfsetstrokecolor{eps2pgf_color}\pgfsetfillcolor{eps2pgf_color}
\pgfpathmoveto{\pgfqpoint{0.117cm}{1.421cm}}
\pgfpathlineto{\pgfqpoint{0.682cm}{0.671cm}}
\pgfpathlineto{\pgfqpoint{1.246cm}{1.421cm}}
\pgfusepath{stroke}
\end{pgfscope}
\definecolor{eps2pgf_color}{gray}{0}\pgfsetstrokecolor{eps2pgf_color}\pgfsetfillcolor{eps2pgf_color}
\pgfpathmoveto{\pgfqpoint{0.273cm}{1.395cm}}
\pgfpathcurveto{\pgfqpoint{0.273cm}{1.432cm}}{\pgfqpoint{0.259cm}{1.467cm}}{\pgfqpoint{0.233cm}{1.492cm}}
\pgfpathcurveto{\pgfqpoint{0.207cm}{1.518cm}}{\pgfqpoint{0.173cm}{1.532cm}}{\pgfqpoint{0.137cm}{1.532cm}}
\pgfpathcurveto{\pgfqpoint{0.1cm}{1.532cm}}{\pgfqpoint{0.066cm}{1.518cm}}{\pgfqpoint{0.04cm}{1.492cm}}
\pgfpathcurveto{\pgfqpoint{0.014cm}{1.467cm}}{\pgfqpoint{0cm}{1.432cm}}{\pgfqpoint{0cm}{1.395cm}}
\pgfpathcurveto{\pgfqpoint{0cm}{1.359cm}}{\pgfqpoint{0.014cm}{1.324cm}}{\pgfqpoint{0.04cm}{1.299cm}}
\pgfpathcurveto{\pgfqpoint{0.066cm}{1.273cm}}{\pgfqpoint{0.1cm}{1.258cm}}{\pgfqpoint{0.137cm}{1.258cm}}
\pgfpathcurveto{\pgfqpoint{0.173cm}{1.258cm}}{\pgfqpoint{0.207cm}{1.273cm}}{\pgfqpoint{0.233cm}{1.299cm}}
\pgfpathcurveto{\pgfqpoint{0.259cm}{1.324cm}}{\pgfqpoint{0.273cm}{1.359cm}}{\pgfqpoint{0.273cm}{1.395cm}}
\pgfusepath{fill}
\begin{pgfscope}
\pgfsetdash{}{0cm}
\pgfsetlinewidth{0.818mm}
\pgfsetmiterlimit{7.0}
\pgfpathmoveto{\pgfqpoint{0.682cm}{0.671cm}}
\pgfpathlineto{\pgfqpoint{0.679cm}{1.418cm}}
\pgfusepath{stroke}
\end{pgfscope}
\pgfpathmoveto{\pgfqpoint{0.815cm}{1.399cm}}
\pgfpathcurveto{\pgfqpoint{0.815cm}{1.435cm}}{\pgfqpoint{0.801cm}{1.47cm}}{\pgfqpoint{0.775cm}{1.496cm}}
\pgfpathcurveto{\pgfqpoint{0.75cm}{1.521cm}}{\pgfqpoint{0.715cm}{1.536cm}}{\pgfqpoint{0.679cm}{1.536cm}}
\pgfpathcurveto{\pgfqpoint{0.643cm}{1.536cm}}{\pgfqpoint{0.608cm}{1.521cm}}{\pgfqpoint{0.582cm}{1.496cm}}
\pgfpathcurveto{\pgfqpoint{0.557cm}{1.47cm}}{\pgfqpoint{0.542cm}{1.435cm}}{\pgfqpoint{0.542cm}{1.399cm}}
\pgfpathcurveto{\pgfqpoint{0.542cm}{1.363cm}}{\pgfqpoint{0.557cm}{1.328cm}}{\pgfqpoint{0.582cm}{1.302cm}}
\pgfpathcurveto{\pgfqpoint{0.608cm}{1.276cm}}{\pgfqpoint{0.643cm}{1.262cm}}{\pgfqpoint{0.679cm}{1.262cm}}
\pgfpathcurveto{\pgfqpoint{0.715cm}{1.262cm}}{\pgfqpoint{0.75cm}{1.276cm}}{\pgfqpoint{0.775cm}{1.302cm}}
\pgfpathcurveto{\pgfqpoint{0.801cm}{1.328cm}}{\pgfqpoint{0.815cm}{1.363cm}}{\pgfqpoint{0.815cm}{1.399cm}}
\pgfusepath{fill}
\pgfpathmoveto{\pgfqpoint{1.345cm}{1.371cm}}
\pgfpathcurveto{\pgfqpoint{1.345cm}{1.408cm}}{\pgfqpoint{1.331cm}{1.442cm}}{\pgfqpoint{1.305cm}{1.468cm}}
\pgfpathcurveto{\pgfqpoint{1.28cm}{1.494cm}}{\pgfqpoint{1.245cm}{1.508cm}}{\pgfqpoint{1.209cm}{1.508cm}}
\pgfpathcurveto{\pgfqpoint{1.172cm}{1.508cm}}{\pgfqpoint{1.138cm}{1.494cm}}{\pgfqpoint{1.112cm}{1.468cm}}
\pgfpathcurveto{\pgfqpoint{1.087cm}{1.442cm}}{\pgfqpoint{1.072cm}{1.408cm}}{\pgfqpoint{1.072cm}{1.371cm}}
\pgfpathcurveto{\pgfqpoint{1.072cm}{1.335cm}}{\pgfqpoint{1.087cm}{1.3cm}}{\pgfqpoint{1.112cm}{1.274cm}}
\pgfpathcurveto{\pgfqpoint{1.138cm}{1.249cm}}{\pgfqpoint{1.172cm}{1.234cm}}{\pgfqpoint{1.209cm}{1.234cm}}
\pgfpathcurveto{\pgfqpoint{1.245cm}{1.234cm}}{\pgfqpoint{1.28cm}{1.249cm}}{\pgfqpoint{1.305cm}{1.274cm}}
\pgfpathcurveto{\pgfqpoint{1.331cm}{1.3cm}}{\pgfqpoint{1.345cm}{1.335cm}}{\pgfqpoint{1.345cm}{1.371cm}}
\pgfusepath{fill}
\begin{pgfscope}
\pgfsetdash{}{0cm}
\pgfsetlinewidth{0.818mm}
\pgfsetroundcap
\pgfsetmiterlimit{4.0}
\pgfpathmoveto{\pgfqpoint{0.682cm}{0.671cm}}
\pgfpathlineto{\pgfqpoint{0.682cm}{0.042cm}}
\pgfusepath{stroke}
\end{pgfscope}
\end{pgfscope}
\end{pgfscope}
\end{pgfscope}
\end{tikzpicture}}}
  + \phi_{M, \varepsilon} = - \LL_{\varepsilon}^{- 1} [3\lambda
  (\UU_{>}^{\varepsilon} \llbracket X^2_{M, \varepsilon} \rrbracket \succ
  Y_{M, \varepsilon}] + \phi_{M, \varepsilon}. \label{eq:24z}
\end{equation}

\begin{theorem}
  \label{thm:phitight}Let $\beta \in (0, 1 / 4)$. Then  for
  all $p \in [1, \infty)$ and $\tau \in (0, T)$
  \[ \sup_{\varepsilon \in \mathcal{A}, M > 0} \mathbb{E} \| \varphi_{M,
     \varepsilon} \|^{2 p}_{W^{\beta, 1}_T B_{1, 1}^{- 1 - 3 \kappa,\varepsilon} (\rho^{4
     + \sigma})} + \sup_{\varepsilon \in \mathcal{A}, M > 0} \mathbb{E} \|
     \varphi_{M, \varepsilon} \|^{2 p}_{L^{\infty}_{\tau, T} H^{- 1 / 2 -2
     \kappa,\varepsilon} (\rho^2)} \leqslant C_\lambda < \infty, \]
  where $L^{\infty}_{\tau, T} H^{- 1 / 2 - 2\kappa,\varepsilon} (\rho^2) = L^{\infty}
  (\tau, T ; H^{- 1 / 2 - 2\kappa,\varepsilon} (\rho^2))$.
\end{theorem}

\begin{proof}
  Let us begin with the first bound. According to Proposition~\ref{prop:reg} and
  Theorem~\ref{thm:tight} we obtain that
  \[ \mathbb{E} \| \chi_{M, \varepsilon} \|^{2 p}_{L^1_T B_{1, 1}^{1 + 3
     \kappa, \varepsilon} (\rho^4)} \leqslant C_{T,\lambda} \mathbb{E}Q_{\rho}
     (\mathbb{X}_{M, \varepsilon}) (1 +\mathbb{E}\| \rho^2 \phi_{M,
     \varepsilon} (0)\|^{2 p}_{L^{2, \varepsilon}}) \]
  \[ \leqslant C_{T,\lambda} \mathbb{E}Q_{\rho} (\mathbb{X}_{\varepsilon}) (1
     +\mathbb{E}\| \rho^2 (\varphi_{M, \varepsilon} (0) - X_{M, \varepsilon}
     (0)) \|^{2 p}_{L^{2, \varepsilon}} +\mathbb{E}\| \rho^2 Y_{M,
     \varepsilon} (0)\|^{2 p}_{L^{2, \varepsilon}}) \]
  is bounded uniformly in $M, \varepsilon$. In addition, the computations in
  the proof of Proposition~\ref{prop:reg} imply that also $\mathbb{E} \left\|
  \LL_{\varepsilon} \chi_{M, \varepsilon} \right\|^{2 p}_{L_T^1 B_{1, 1}^{- 1
  + 3 \kappa, \varepsilon} (\rho^4)}$ is bounded uniformly in $M,
  \varepsilon$. As a consequence, we deduce that
  \[ \mathbb{E} \| \partial_t \chi_{M, \varepsilon} \|^{2 p}_{L^1_T B^{- 1 + 3
     \kappa, \varepsilon}_{1, 1} (\rho^4)} \leqslant \mathbb{E} \|
     (\Delta_{\varepsilon} - m^2) \chi_{M, \varepsilon} \|^{2 p}_{L^1_T B^{- 1
     + 3 \kappa, \varepsilon}_{1, 1} (\rho^4)} +\mathbb{E} \left\|
     \LL_{\varepsilon} \chi_{M, \varepsilon} \right\|^{2 p}_{L^1_T B^{- 1 + 3
     \kappa, \varepsilon}_{1, 1} (\rho^4)} \]
  is also bounded uniformly in $M, \varepsilon$.
  
  Next, we apply a similar approach to derive uniform time regularity of
  $\phi_{M, \varepsilon}$. To this end, we study the right hand side of
  {\eqref{eq:phiU}}. Observe that due to the energy estimate from Theorem~\ref{th:energy-estimate} and the bound from Proposition~\ref{prop:reg} together
  with Theorem~\ref{thm:tight} the following are bounded uniformly in $M,
  \varepsilon$
  \[ \mathbb{E} \| \llbracket X_{M, \varepsilon}^2 \rrbracket \succ \phi_{M,
     \varepsilon} \|^{2 p}_{L^2_T H^{- 1 - \kappa, \varepsilon} (\rho^{2 +
     \sigma})}, \hspace{1em} \mathbb{E} \| \llbracket X_{M, \varepsilon}^2
     \rrbracket \circ \chi_{M, \varepsilon} \|^{2 p}_{L^1_T B^{2 \kappa,\varepsilon}_{1,
     1} (\rho^{4 + \sigma})}, \]
  whereas all the other terms on the right hand side of {\eqref{eq:phiU}} are
  uniformly bounded in better function spaces. Hence we deduce that
  \[ \mathbb{E} \| \partial_t \phi_{M, \varepsilon} \|^{2 p}_{L_T^1 B^{- 1 - 3
     \kappa, \varepsilon}_{1, 1} (\rho^{4 + \sigma})} \leqslant \mathbb{E} \|
     (\Delta_{\varepsilon} - m^2) \phi_{M, \varepsilon} \|^{2 p}_{L^1_T B^{- 1
     - 3 \kappa, \varepsilon}_{1, 1} (\rho^{4 + \sigma})} +\mathbb{E} \left\|
     \LL_{\varepsilon} \phi_{M, \varepsilon} \right\|^{2 p}_{L^1_T B^{- 1 - 3
     \kappa, \varepsilon}_{1, 1} (\rho^{4 + \sigma})} \]
  is bounded uniformly in $M, \varepsilon$.
  
  Now we have all in hand to derive a uniform time regularity of $\zeta_{M,
  \varepsilon}$. Using Schauder estimates together with {\eqref{eq:24z}} it
  holds that
  \[ \mathbb{E} \| \zeta_{M, \varepsilon} \|^{2 p}_{W^{(1 - 2 \kappa) / 2,
     1}_T B_{1, 1}^{- 1 - 3 \kappa, \varepsilon} (\rho^{4 + \sigma})}
     \leqslant \mathbb{E} \left\| \LL_{\varepsilon}^{- 1} [3\lambda
     (\UU_{>}^{\varepsilon} \llbracket X^2_{M, \varepsilon} \rrbracket \succ
     Y_{M, \varepsilon}] \right\|^{2 p}_{C^{(1 - \kappa) / 2}_T L^{\infty,\varepsilon}
     (\rho^{\sigma})} \]
  \[ +\mathbb{E} \| \phi_{M, \varepsilon} \|^{2 p}_{W_T^{1, 1} B^{- 1 - 3
     \kappa, \varepsilon}_{1, 1} (\rho^{4 + \sigma})} \]
  is bounded uniformly in $M, \varepsilon$.
  
  Finally, since for all $\beta \in (0, 1)$ we have that both
  \[ \mathbb{E} \| X_{M, \varepsilon} \|^{2 p}_{C^{\beta}_T \CC^{- 1 / 2 -
     \kappa - 2 \beta,\varepsilon} (\rho^{\sigma})}, \hspace{1em} \mathbb{E} \|
     X^{\!\resizebox{0.6em}{!}{
\begin{tikzpicture}
\pgfpathmoveto{\pgfqpoint{0cm}{-0.035cm}}
\pgfpathlineto{\pgfqpoint{1.376cm}{-0.035cm}}
\pgfpathlineto{\pgfqpoint{1.376cm}{1.552cm}}
\pgfpathlineto{\pgfqpoint{0cm}{1.552cm}}
\pgfpathclose
\pgfusepath{clip}
\begin{pgfscope}
\begin{pgfscope}
\pgfpathmoveto{\pgfqpoint{0cm}{-0.035cm}}
\pgfpathlineto{\pgfqpoint{1.376cm}{-0.035cm}}
\pgfpathlineto{\pgfqpoint{1.376cm}{1.552cm}}
\pgfpathlineto{\pgfqpoint{0cm}{1.552cm}}
\pgfpathclose
\pgfusepath{clip}
\begin{pgfscope}
\begin{pgfscope}
\pgfsetdash{}{0cm}
\pgfsetlinewidth{0.818mm}
\pgfsetroundcap
\pgfsetroundjoin
\pgfsetmiterlimit{7.0}
\definecolor{eps2pgf_color}{gray}{0}\pgfsetstrokecolor{eps2pgf_color}\pgfsetfillcolor{eps2pgf_color}
\pgfpathmoveto{\pgfqpoint{0.117cm}{1.421cm}}
\pgfpathlineto{\pgfqpoint{0.682cm}{0.671cm}}
\pgfpathlineto{\pgfqpoint{1.246cm}{1.421cm}}
\pgfusepath{stroke}
\end{pgfscope}
\definecolor{eps2pgf_color}{gray}{0}\pgfsetstrokecolor{eps2pgf_color}\pgfsetfillcolor{eps2pgf_color}
\pgfpathmoveto{\pgfqpoint{0.273cm}{1.395cm}}
\pgfpathcurveto{\pgfqpoint{0.273cm}{1.432cm}}{\pgfqpoint{0.259cm}{1.467cm}}{\pgfqpoint{0.233cm}{1.492cm}}
\pgfpathcurveto{\pgfqpoint{0.207cm}{1.518cm}}{\pgfqpoint{0.173cm}{1.532cm}}{\pgfqpoint{0.137cm}{1.532cm}}
\pgfpathcurveto{\pgfqpoint{0.1cm}{1.532cm}}{\pgfqpoint{0.066cm}{1.518cm}}{\pgfqpoint{0.04cm}{1.492cm}}
\pgfpathcurveto{\pgfqpoint{0.014cm}{1.467cm}}{\pgfqpoint{0cm}{1.432cm}}{\pgfqpoint{0cm}{1.395cm}}
\pgfpathcurveto{\pgfqpoint{0cm}{1.359cm}}{\pgfqpoint{0.014cm}{1.324cm}}{\pgfqpoint{0.04cm}{1.299cm}}
\pgfpathcurveto{\pgfqpoint{0.066cm}{1.273cm}}{\pgfqpoint{0.1cm}{1.258cm}}{\pgfqpoint{0.137cm}{1.258cm}}
\pgfpathcurveto{\pgfqpoint{0.173cm}{1.258cm}}{\pgfqpoint{0.207cm}{1.273cm}}{\pgfqpoint{0.233cm}{1.299cm}}
\pgfpathcurveto{\pgfqpoint{0.259cm}{1.324cm}}{\pgfqpoint{0.273cm}{1.359cm}}{\pgfqpoint{0.273cm}{1.395cm}}
\pgfusepath{fill}
\begin{pgfscope}
\pgfsetdash{}{0cm}
\pgfsetlinewidth{0.818mm}
\pgfsetmiterlimit{7.0}
\pgfpathmoveto{\pgfqpoint{0.682cm}{0.671cm}}
\pgfpathlineto{\pgfqpoint{0.679cm}{1.418cm}}
\pgfusepath{stroke}
\end{pgfscope}
\pgfpathmoveto{\pgfqpoint{0.815cm}{1.399cm}}
\pgfpathcurveto{\pgfqpoint{0.815cm}{1.435cm}}{\pgfqpoint{0.801cm}{1.47cm}}{\pgfqpoint{0.775cm}{1.496cm}}
\pgfpathcurveto{\pgfqpoint{0.75cm}{1.521cm}}{\pgfqpoint{0.715cm}{1.536cm}}{\pgfqpoint{0.679cm}{1.536cm}}
\pgfpathcurveto{\pgfqpoint{0.643cm}{1.536cm}}{\pgfqpoint{0.608cm}{1.521cm}}{\pgfqpoint{0.582cm}{1.496cm}}
\pgfpathcurveto{\pgfqpoint{0.557cm}{1.47cm}}{\pgfqpoint{0.542cm}{1.435cm}}{\pgfqpoint{0.542cm}{1.399cm}}
\pgfpathcurveto{\pgfqpoint{0.542cm}{1.363cm}}{\pgfqpoint{0.557cm}{1.328cm}}{\pgfqpoint{0.582cm}{1.302cm}}
\pgfpathcurveto{\pgfqpoint{0.608cm}{1.276cm}}{\pgfqpoint{0.643cm}{1.262cm}}{\pgfqpoint{0.679cm}{1.262cm}}
\pgfpathcurveto{\pgfqpoint{0.715cm}{1.262cm}}{\pgfqpoint{0.75cm}{1.276cm}}{\pgfqpoint{0.775cm}{1.302cm}}
\pgfpathcurveto{\pgfqpoint{0.801cm}{1.328cm}}{\pgfqpoint{0.815cm}{1.363cm}}{\pgfqpoint{0.815cm}{1.399cm}}
\pgfusepath{fill}
\pgfpathmoveto{\pgfqpoint{1.345cm}{1.371cm}}
\pgfpathcurveto{\pgfqpoint{1.345cm}{1.408cm}}{\pgfqpoint{1.331cm}{1.442cm}}{\pgfqpoint{1.305cm}{1.468cm}}
\pgfpathcurveto{\pgfqpoint{1.28cm}{1.494cm}}{\pgfqpoint{1.245cm}{1.508cm}}{\pgfqpoint{1.209cm}{1.508cm}}
\pgfpathcurveto{\pgfqpoint{1.172cm}{1.508cm}}{\pgfqpoint{1.138cm}{1.494cm}}{\pgfqpoint{1.112cm}{1.468cm}}
\pgfpathcurveto{\pgfqpoint{1.087cm}{1.442cm}}{\pgfqpoint{1.072cm}{1.408cm}}{\pgfqpoint{1.072cm}{1.371cm}}
\pgfpathcurveto{\pgfqpoint{1.072cm}{1.335cm}}{\pgfqpoint{1.087cm}{1.3cm}}{\pgfqpoint{1.112cm}{1.274cm}}
\pgfpathcurveto{\pgfqpoint{1.138cm}{1.249cm}}{\pgfqpoint{1.172cm}{1.234cm}}{\pgfqpoint{1.209cm}{1.234cm}}
\pgfpathcurveto{\pgfqpoint{1.245cm}{1.234cm}}{\pgfqpoint{1.28cm}{1.249cm}}{\pgfqpoint{1.305cm}{1.274cm}}
\pgfpathcurveto{\pgfqpoint{1.331cm}{1.3cm}}{\pgfqpoint{1.345cm}{1.335cm}}{\pgfqpoint{1.345cm}{1.371cm}}
\pgfusepath{fill}
\begin{pgfscope}
\pgfsetdash{}{0cm}
\pgfsetlinewidth{0.818mm}
\pgfsetroundcap
\pgfsetmiterlimit{4.0}
\pgfpathmoveto{\pgfqpoint{0.682cm}{0.671cm}}
\pgfpathlineto{\pgfqpoint{0.682cm}{0.042cm}}
\pgfusepath{stroke}
\end{pgfscope}
\end{pgfscope}
\end{pgfscope}
\end{pgfscope}
\end{tikzpicture}}}_{M, \varepsilon} \|^{2 p}_{C^{\beta}_T \CC^{1 / 2 -
     \kappa - 2 \beta,\varepsilon} (\rho^{\sigma})} \]
  are bounded uniformly in $M, \varepsilon$, we conclude that so is
  $\mathbb{E} \| \varphi_{M, \varepsilon} \|^{2 p}_{W^{\beta, 1}_T B_{1, 1}^{-
  1 - 3 \kappa,\varepsilon} (\rho^{4 + \sigma})}$ for $\beta \in (0, 1 / 4)$, which
  completes the proof of the first bound.
  
  In order to establish the second bound we recall the decomposition
  $\varphi_{M, \varepsilon} = X_{M, \varepsilon} + Y_{M, \varepsilon} +
  \phi_{M, \varepsilon}$ and make use of the energy estimate from Corollary~\ref{cor:Lp}. Taking supremum over $t \in [\tau, T]$ and expectation implies
  \[ \sup_{\varepsilon \in \mathcal{A}, M > 0} \mathbb{E} \| \phi_{M,
     \varepsilon} \|^{2 p}_{L^{\infty}_{\tau, T} L^{2,\varepsilon} (\rho^2)} < \infty . \]
  The claim now follows using the bound for $X_{M, \varepsilon}$ together with
  the bound for $Y_{M, \varepsilon}$ in Lemma~\ref{lem:Y1}.
\end{proof}

Even though the uniform bound in the previous result is far from being
  optimal, it is sufficient for our purposes below.

\begin{corollary}
  \label{cor:t}Let $\rho$ be a weight such that $\rho^{\iota} \in L^4$ for
  some $\iota \in (0, 1)$. Let $\beta \in (0, 1 / 4)$ and $\alpha \in (0,
  \beta)$. Then the family of joint laws of $(\mathcal{E}^{\varepsilon}
  \varphi_{M, \varepsilon}, \mathcal{E}^{\varepsilon} \mathbb{X}_{M,
  \varepsilon})$ is tight on $W^{\alpha, 1}_{\tmop{loc}} B_{1, 1}^{- 1 - 4
  \kappa} (\rho^{4 + \sigma}) \times C^{\kappa / 2}_{\tmop{loc}}
  \mathcal{X}^{}$, where
  \[ \mathcal{X} \assign \prod_{i = 1, \ldots, 7} \CC^{\alpha (i) - \kappa}
     (\rho^{\sigma}) \]
  with $\alpha (1) = \alpha (7) = - 1 / 2,$ $\alpha (2) = - 1,$ $\alpha (3) =
  1 / 2,$ $\alpha (4) = \alpha (5) = \alpha (6) = 0$.
\end{corollary}

\begin{proof}
  According to Theorem 6.31 in {\cite{T06}} we have the compact embedding
  \[ B_{1, 1}^{- 1 - 3 \kappa} (\rho^{4 + \sigma}) \subset B_{1, 1}^{- 1 - 4
     \kappa} (\rho^{4 + 2 \sigma}) \]
  and consequently since $\alpha < \beta$ the embedding
  \[ W^{\beta, 1}_{\tmop{loc}} B_{1, 1}^{- 1 - 3 \kappa} (\rho^{4 + \sigma})
     \subset W^{\alpha, 1}_{\tmop{loc}} B_{1, 1}^{- 1 - 4 \kappa} (\rho^{4 + 2
     \sigma}) \]
  is compact, see e.g. Theorem 5.1 {\cite{A00}}. Hence the desired tightness
  of $\mathcal{E}^{\varepsilon} \varphi_{M, \varepsilon}$ follows from Theorem~\ref{thm:phitight} and Lemma~\ref{lem:ext}. The tightness of
  $\mathcal{E}^{\varepsilon} \mathbb{X}_{M, \varepsilon}$ follows from the
  usual arguments and does not pose any problems.
\end{proof}

As a consequence, we may extract a converging subsequence of the joint laws of
the  processes $(\mathcal{E}^{\varepsilon} \varphi_{M, \varepsilon},
\mathcal{E}^{\varepsilon} \mathbb{X}_{M, \varepsilon})_{M, \varepsilon}$ in
$W^{\alpha, 1}_{\tmop{loc}} B_{1, 1}^{- 1 - 4 \kappa} (\rho^{4 + \sigma})
\times C^{\kappa / 2}_{\tmop{loc}} \mathcal{X}^{}$. Let $\hat{\mu}$ denote any
limit point. We recall that $\mathbb{X}_{M,\varepsilon}$ denotes the collection of all the necessary stochastic objects, see \eqref{eq:XX}. We denote by $(\varphi, \mathbb{X})$ the canonical process on
$W^{\alpha, 1}_{\tmop{loc}} B_{1, 1}^{- 1 - 4 \kappa} (\rho^{4 + \sigma})
\times C^{\kappa / 2}_{\tmop{loc}} \mathcal{X}^{}$  and let $\mu$ be the law
of the pair $(\varphi, X)$ under $\hat{\mu}$ (i.e. the projection of $\hat{\mu}$
to the first two components). Observe that there exists a measurable map $\Psi :
(\varphi, X) \mapsto (\varphi, \mathbb{X})$ such that $\hat{\mu} = \mu \circ
\Psi^{- 1}$. Therefore we can represent expectations under $\hat{\mu}$ as
expectations under $\mu$ with the understanding that the elements of
$\mathbb{X}$ are constructed canonically from $X$ via $\Psi$. Furthermore, $Y,\phi,\zeta,\chi$ are defined analogously as on the approximate level as measurable functions of the pair $(\varphi,X)$. In particular, the limit localizer $\UU_{>}$ is determined by the constant $L_{0}$ obtained in Lemma \ref{lem:Y1}. Consequently, all the above uniform estimates are preserved for the limiting measure and the convergence of the corresponding lattice approximations to $Y,\phi,\zeta,\chi$ follows. In addition, the limiting process $\varphi$ is stationary in the following distributional sense: for all $f\in C^{\infty}_{c}(\mathbb{R}_{+})$ and all $\tau>0$, the laws of
$$
\varphi(f)\quad\text{and}\quad \varphi(f(\cdot-\tau))\quad\text{on}\quad \mathcal{S}'(\mathbb{R}^{3})
$$
coincide. Based on the time regularity of $\varphi$ it can be shown that this implies that the laws of $\varphi(t)$ and $\varphi(t+\tau)$ coincide for all $\tau>0$ and a.e. $t\in[0,\infty)$. The projection of $\mu$ on $\varphi(t)$ taken from this set of full measure is the measure $\nu$ as obtained in Theorem~\ref{thm:main}.

\subsection{Integration by parts formula}

The goal of his section is to derive an integration by parts formula for the
$\Phi^4_3$ measure on the full space. To this end, we begin with the
corresponding integration by parts formula on the approximate level, that is,
for the measures $\nu_{M, \varepsilon}$ and pass to the limit.

Let $F$ be a cylinder functional on $\mathcal{S}' (\mathbb{R}^3)$, that is, $F
(\varphi) = \Phi (\varphi (f_1), \ldots, \varphi (f_n))$ for some polynomial $\Phi :
\mathbb{R}^n \rightarrow \mathbb{R}$ and $f_1, \ldots, f_n \in \mathcal{S}
(\mathbb{R}^3)$. Let $\mathD F (\varphi)$ denote the $L^2$-gradient of $F$.
Then  for fields $\varphi_{\varepsilon}$ defined on
$\Lambda_{\varepsilon}$ we have
\[ \frac{\partial F (\mathcal{E}^{\varepsilon}
   \varphi_{\varepsilon})}{\partial \varphi_{\varepsilon} (x)} = \varepsilon^d
   \sum_{i = 1}^n \partial_i \Phi ((\mathcal{E}^{\varepsilon}
   \varphi_{\varepsilon}) (f_1), \ldots, (\mathcal{E}^{\varepsilon}
   \varphi_{\varepsilon}) (f_n)) (w_{\varepsilon} \ast f_i) (x) =
   \varepsilon^d [w_{\varepsilon} \ast \mathD F (\mathcal{E}^{\varepsilon}
   \varphi_{\varepsilon})] (x), \]
where $x\in \Lambda_{\varepsilon}$ and $w_{\varepsilon}$ is the kernel involved in the definition of the
extension operator $\mathcal{E}^{\varepsilon}$ from Section~\ref{s:ext}. By integration by parts it
follows that
\[ \int [w_{\varepsilon} \ast \mathD F (\mathcal{E}^{\varepsilon} \varphi)]
   (x) \nu_{M, \varepsilon} (\mathd \varphi) = \frac{1}{\varepsilon^d} \int
   \frac{\partial F (\mathcal{E}^{\varepsilon} \varphi)}{\partial \varphi (x)}
   \nu_{M, \varepsilon} (\mathd \varphi) = \frac{2}{\varepsilon^d} \int F
   (\mathcal{E}^{\varepsilon} \varphi) \frac{\partial V_{M, \varepsilon}
   (\varphi)}{\partial \varphi (x)} \nu_{M, \varepsilon} (\mathd \varphi) \]
\begin{equation}
  = 2 \int F (\mathcal{E}^{\varepsilon} \varphi) [\lambda \varphi (x)^3 + (- 3\lambda a_{M,
  \varepsilon} + 3\lambda^2 b_{M, \varepsilon}) \varphi (x)] \nu_{M, \varepsilon}
  (\mathd \varphi) + 2 \int F (\mathcal{E}^{\varepsilon} \varphi) [m^2 -
  \Delta_{\varepsilon}] \varphi (x) \nu_{M, \varepsilon} (\mathd \varphi) .
  \label{eq:ibp1}
\end{equation}
According to Theorem~\ref{thm:main}, we can already pass to the limit on the
left hand side as well as in the second term on the right hand side of
{\eqref{eq:ibp1}}. Namely, we obtain for any accumulation point $\nu$ and any
(relabeled) subsequence $(\nu_{M, \varepsilon} \circ
(\mathcal{E}^{\varepsilon})^{- 1})_{M, \varepsilon}$ converging to $\nu$ that
the following convergences hold in the sense of distributions in the variable
$x \in \mathbb{R}^3$
\[ \int \mathcal{E}^{\varepsilon} [w_{\varepsilon} \ast \mathD F
   (\mathcal{E}^{\varepsilon} \varphi)] (x) \nu_{M, \varepsilon} (\mathd
   \varphi) \rightarrow \int \mathD F (\mathcal{E}^{\varepsilon} \varphi) (x)
   \nu (\mathd \varphi), \]
\[ \int F (\mathcal{E}^{\varepsilon} \varphi) \mathcal{E}^{\varepsilon} [m^2 -
   \Delta_{\varepsilon}] \varphi (x) \nu_{M, \varepsilon} (\mathd \varphi)
   \rightarrow \int F (\varphi) [m^2 - \Delta] \varphi (x) \nu (\mathd
   \varphi) . \]
The remainder of this section is devoted to the passage to the limit in
{\eqref{eq:ibp1}}, leading to the integration by parts formula for the
limiting measure in Theorem~\ref{thm:ibp} below. In particular, it is
necessary to find a way to control the convergence of the cubic term and to
interpret the limit under the $\Phi^4_3$ measure.

Let us denote
\[ \llbracket \varphi^3 \rrbracket_{M, \varepsilon} (y) \assign \varphi (y)^3
   + (- 3 a_{M, \varepsilon} + 3\lambda b_{M, \varepsilon}) \varphi (y) . \]
We shall analyze carefully the distributions $\mathcal{J}_{M, \varepsilon} (F)
\in \mathcal{S}' (\Lambda_{\varepsilon})$ given by
\[ \mathcal{J}_{M, \varepsilon} (F) \assign x \mapsto \int F
   (\mathcal{E}^{\varepsilon} \varphi) \llbracket \varphi^3 \rrbracket_{M,
   \varepsilon} (x) \nu_{M, \varepsilon} (\mathd \varphi), \]
in order to determine the limit of $\mathcal{E}^{\varepsilon} \mathcal{J}_{M,
\varepsilon} (F)$ (as a distribution in $x \in \mathbb{R}^3$) as $(M,
\varepsilon) \rightarrow (\infty, 0)$. Unfortunately, even for the Gaussian
case when $\lambda = 0$ one cannot give a well-defined meaning to the random
variable $\varphi^3$ under the measure $\nu$. Additive renormalization is not
enough to cure this problem since it is easy to see that the variance of the
putative Wick renormalized limiting field
\[ \llbracket \varphi^3 \rrbracket = \lim_{\varepsilon \rightarrow 0, M
   \rightarrow \infty} \mathcal{E}^{\varepsilon} \llbracket \varphi^3
   \rrbracket_{M, \varepsilon} \]
is infinite. In the best of the cases one can hope that the renormalized cube
$\llbracket \varphi^3 \rrbracket$ makes sense once integrated against smooth
cylinder functions $F (\varphi)$. Otherwise stated, one could try to prove
that $(\mathcal{J}_{M, \varepsilon})_{M, \varepsilon}$ converges as a linear
functional on cylinder test functions over $\mathcal{S}' (\mathbb{R}^3)$.

To this end, we work with the stationary solution $\varphi_{M, \varepsilon}$
and introduce the additional notation
\[ \llbracket \varphi_{M, \varepsilon}^3 \rrbracket (t, y) \assign \varphi_{M,
   \varepsilon} (t, y)^3 + (- 3 a_{M, \varepsilon} + 3\lambda b_{M, \varepsilon})
   \varphi_{M, \varepsilon} (t, y) . \]
As the next step, we employ the decomposition
\[ \varphi_{M, \varepsilon} = X_{M, \varepsilon} -\lambda X_{M,
   \varepsilon}^{\!\resizebox{0.6em}{!}{
\begin{tikzpicture}
\pgfpathmoveto{\pgfqpoint{0cm}{-0.035cm}}
\pgfpathlineto{\pgfqpoint{1.376cm}{-0.035cm}}
\pgfpathlineto{\pgfqpoint{1.376cm}{1.552cm}}
\pgfpathlineto{\pgfqpoint{0cm}{1.552cm}}
\pgfpathclose
\pgfusepath{clip}
\begin{pgfscope}
\begin{pgfscope}
\pgfpathmoveto{\pgfqpoint{0cm}{-0.035cm}}
\pgfpathlineto{\pgfqpoint{1.376cm}{-0.035cm}}
\pgfpathlineto{\pgfqpoint{1.376cm}{1.552cm}}
\pgfpathlineto{\pgfqpoint{0cm}{1.552cm}}
\pgfpathclose
\pgfusepath{clip}
\begin{pgfscope}
\begin{pgfscope}
\pgfsetdash{}{0cm}
\pgfsetlinewidth{0.818mm}
\pgfsetroundcap
\pgfsetroundjoin
\pgfsetmiterlimit{7.0}
\definecolor{eps2pgf_color}{gray}{0}\pgfsetstrokecolor{eps2pgf_color}\pgfsetfillcolor{eps2pgf_color}
\pgfpathmoveto{\pgfqpoint{0.117cm}{1.421cm}}
\pgfpathlineto{\pgfqpoint{0.682cm}{0.671cm}}
\pgfpathlineto{\pgfqpoint{1.246cm}{1.421cm}}
\pgfusepath{stroke}
\end{pgfscope}
\definecolor{eps2pgf_color}{gray}{0}\pgfsetstrokecolor{eps2pgf_color}\pgfsetfillcolor{eps2pgf_color}
\pgfpathmoveto{\pgfqpoint{0.273cm}{1.395cm}}
\pgfpathcurveto{\pgfqpoint{0.273cm}{1.432cm}}{\pgfqpoint{0.259cm}{1.467cm}}{\pgfqpoint{0.233cm}{1.492cm}}
\pgfpathcurveto{\pgfqpoint{0.207cm}{1.518cm}}{\pgfqpoint{0.173cm}{1.532cm}}{\pgfqpoint{0.137cm}{1.532cm}}
\pgfpathcurveto{\pgfqpoint{0.1cm}{1.532cm}}{\pgfqpoint{0.066cm}{1.518cm}}{\pgfqpoint{0.04cm}{1.492cm}}
\pgfpathcurveto{\pgfqpoint{0.014cm}{1.467cm}}{\pgfqpoint{0cm}{1.432cm}}{\pgfqpoint{0cm}{1.395cm}}
\pgfpathcurveto{\pgfqpoint{0cm}{1.359cm}}{\pgfqpoint{0.014cm}{1.324cm}}{\pgfqpoint{0.04cm}{1.299cm}}
\pgfpathcurveto{\pgfqpoint{0.066cm}{1.273cm}}{\pgfqpoint{0.1cm}{1.258cm}}{\pgfqpoint{0.137cm}{1.258cm}}
\pgfpathcurveto{\pgfqpoint{0.173cm}{1.258cm}}{\pgfqpoint{0.207cm}{1.273cm}}{\pgfqpoint{0.233cm}{1.299cm}}
\pgfpathcurveto{\pgfqpoint{0.259cm}{1.324cm}}{\pgfqpoint{0.273cm}{1.359cm}}{\pgfqpoint{0.273cm}{1.395cm}}
\pgfusepath{fill}
\begin{pgfscope}
\pgfsetdash{}{0cm}
\pgfsetlinewidth{0.818mm}
\pgfsetmiterlimit{7.0}
\pgfpathmoveto{\pgfqpoint{0.682cm}{0.671cm}}
\pgfpathlineto{\pgfqpoint{0.679cm}{1.418cm}}
\pgfusepath{stroke}
\end{pgfscope}
\pgfpathmoveto{\pgfqpoint{0.815cm}{1.399cm}}
\pgfpathcurveto{\pgfqpoint{0.815cm}{1.435cm}}{\pgfqpoint{0.801cm}{1.47cm}}{\pgfqpoint{0.775cm}{1.496cm}}
\pgfpathcurveto{\pgfqpoint{0.75cm}{1.521cm}}{\pgfqpoint{0.715cm}{1.536cm}}{\pgfqpoint{0.679cm}{1.536cm}}
\pgfpathcurveto{\pgfqpoint{0.643cm}{1.536cm}}{\pgfqpoint{0.608cm}{1.521cm}}{\pgfqpoint{0.582cm}{1.496cm}}
\pgfpathcurveto{\pgfqpoint{0.557cm}{1.47cm}}{\pgfqpoint{0.542cm}{1.435cm}}{\pgfqpoint{0.542cm}{1.399cm}}
\pgfpathcurveto{\pgfqpoint{0.542cm}{1.363cm}}{\pgfqpoint{0.557cm}{1.328cm}}{\pgfqpoint{0.582cm}{1.302cm}}
\pgfpathcurveto{\pgfqpoint{0.608cm}{1.276cm}}{\pgfqpoint{0.643cm}{1.262cm}}{\pgfqpoint{0.679cm}{1.262cm}}
\pgfpathcurveto{\pgfqpoint{0.715cm}{1.262cm}}{\pgfqpoint{0.75cm}{1.276cm}}{\pgfqpoint{0.775cm}{1.302cm}}
\pgfpathcurveto{\pgfqpoint{0.801cm}{1.328cm}}{\pgfqpoint{0.815cm}{1.363cm}}{\pgfqpoint{0.815cm}{1.399cm}}
\pgfusepath{fill}
\pgfpathmoveto{\pgfqpoint{1.345cm}{1.371cm}}
\pgfpathcurveto{\pgfqpoint{1.345cm}{1.408cm}}{\pgfqpoint{1.331cm}{1.442cm}}{\pgfqpoint{1.305cm}{1.468cm}}
\pgfpathcurveto{\pgfqpoint{1.28cm}{1.494cm}}{\pgfqpoint{1.245cm}{1.508cm}}{\pgfqpoint{1.209cm}{1.508cm}}
\pgfpathcurveto{\pgfqpoint{1.172cm}{1.508cm}}{\pgfqpoint{1.138cm}{1.494cm}}{\pgfqpoint{1.112cm}{1.468cm}}
\pgfpathcurveto{\pgfqpoint{1.087cm}{1.442cm}}{\pgfqpoint{1.072cm}{1.408cm}}{\pgfqpoint{1.072cm}{1.371cm}}
\pgfpathcurveto{\pgfqpoint{1.072cm}{1.335cm}}{\pgfqpoint{1.087cm}{1.3cm}}{\pgfqpoint{1.112cm}{1.274cm}}
\pgfpathcurveto{\pgfqpoint{1.138cm}{1.249cm}}{\pgfqpoint{1.172cm}{1.234cm}}{\pgfqpoint{1.209cm}{1.234cm}}
\pgfpathcurveto{\pgfqpoint{1.245cm}{1.234cm}}{\pgfqpoint{1.28cm}{1.249cm}}{\pgfqpoint{1.305cm}{1.274cm}}
\pgfpathcurveto{\pgfqpoint{1.331cm}{1.3cm}}{\pgfqpoint{1.345cm}{1.335cm}}{\pgfqpoint{1.345cm}{1.371cm}}
\pgfusepath{fill}
\begin{pgfscope}
\pgfsetdash{}{0cm}
\pgfsetlinewidth{0.818mm}
\pgfsetroundcap
\pgfsetmiterlimit{4.0}
\pgfpathmoveto{\pgfqpoint{0.682cm}{0.671cm}}
\pgfpathlineto{\pgfqpoint{0.682cm}{0.042cm}}
\pgfusepath{stroke}
\end{pgfscope}
\end{pgfscope}
\end{pgfscope}
\end{pgfscope}
\end{tikzpicture}}} + \zeta_{M, \varepsilon} \]
in order to find a decomposition that can be controlled by our estimates. We
rewrite
\[ \begin{array}{lll}
     \llbracket \varphi_{M, \varepsilon}^3 \rrbracket & = & \llbracket X_{M,
     \varepsilon}^3 \rrbracket + 3 \llbracket X_{M, \varepsilon}^2 \rrbracket
     (-\lambda X_{M, \varepsilon}^{\!\resizebox{0.6em}{!}{
\begin{tikzpicture}
\pgfpathmoveto{\pgfqpoint{0cm}{-0.035cm}}
\pgfpathlineto{\pgfqpoint{1.376cm}{-0.035cm}}
\pgfpathlineto{\pgfqpoint{1.376cm}{1.552cm}}
\pgfpathlineto{\pgfqpoint{0cm}{1.552cm}}
\pgfpathclose
\pgfusepath{clip}
\begin{pgfscope}
\begin{pgfscope}
\pgfpathmoveto{\pgfqpoint{0cm}{-0.035cm}}
\pgfpathlineto{\pgfqpoint{1.376cm}{-0.035cm}}
\pgfpathlineto{\pgfqpoint{1.376cm}{1.552cm}}
\pgfpathlineto{\pgfqpoint{0cm}{1.552cm}}
\pgfpathclose
\pgfusepath{clip}
\begin{pgfscope}
\begin{pgfscope}
\pgfsetdash{}{0cm}
\pgfsetlinewidth{0.818mm}
\pgfsetroundcap
\pgfsetroundjoin
\pgfsetmiterlimit{7.0}
\definecolor{eps2pgf_color}{gray}{0}\pgfsetstrokecolor{eps2pgf_color}\pgfsetfillcolor{eps2pgf_color}
\pgfpathmoveto{\pgfqpoint{0.117cm}{1.421cm}}
\pgfpathlineto{\pgfqpoint{0.682cm}{0.671cm}}
\pgfpathlineto{\pgfqpoint{1.246cm}{1.421cm}}
\pgfusepath{stroke}
\end{pgfscope}
\definecolor{eps2pgf_color}{gray}{0}\pgfsetstrokecolor{eps2pgf_color}\pgfsetfillcolor{eps2pgf_color}
\pgfpathmoveto{\pgfqpoint{0.273cm}{1.395cm}}
\pgfpathcurveto{\pgfqpoint{0.273cm}{1.432cm}}{\pgfqpoint{0.259cm}{1.467cm}}{\pgfqpoint{0.233cm}{1.492cm}}
\pgfpathcurveto{\pgfqpoint{0.207cm}{1.518cm}}{\pgfqpoint{0.173cm}{1.532cm}}{\pgfqpoint{0.137cm}{1.532cm}}
\pgfpathcurveto{\pgfqpoint{0.1cm}{1.532cm}}{\pgfqpoint{0.066cm}{1.518cm}}{\pgfqpoint{0.04cm}{1.492cm}}
\pgfpathcurveto{\pgfqpoint{0.014cm}{1.467cm}}{\pgfqpoint{0cm}{1.432cm}}{\pgfqpoint{0cm}{1.395cm}}
\pgfpathcurveto{\pgfqpoint{0cm}{1.359cm}}{\pgfqpoint{0.014cm}{1.324cm}}{\pgfqpoint{0.04cm}{1.299cm}}
\pgfpathcurveto{\pgfqpoint{0.066cm}{1.273cm}}{\pgfqpoint{0.1cm}{1.258cm}}{\pgfqpoint{0.137cm}{1.258cm}}
\pgfpathcurveto{\pgfqpoint{0.173cm}{1.258cm}}{\pgfqpoint{0.207cm}{1.273cm}}{\pgfqpoint{0.233cm}{1.299cm}}
\pgfpathcurveto{\pgfqpoint{0.259cm}{1.324cm}}{\pgfqpoint{0.273cm}{1.359cm}}{\pgfqpoint{0.273cm}{1.395cm}}
\pgfusepath{fill}
\begin{pgfscope}
\pgfsetdash{}{0cm}
\pgfsetlinewidth{0.818mm}
\pgfsetmiterlimit{7.0}
\pgfpathmoveto{\pgfqpoint{0.682cm}{0.671cm}}
\pgfpathlineto{\pgfqpoint{0.679cm}{1.418cm}}
\pgfusepath{stroke}
\end{pgfscope}
\pgfpathmoveto{\pgfqpoint{0.815cm}{1.399cm}}
\pgfpathcurveto{\pgfqpoint{0.815cm}{1.435cm}}{\pgfqpoint{0.801cm}{1.47cm}}{\pgfqpoint{0.775cm}{1.496cm}}
\pgfpathcurveto{\pgfqpoint{0.75cm}{1.521cm}}{\pgfqpoint{0.715cm}{1.536cm}}{\pgfqpoint{0.679cm}{1.536cm}}
\pgfpathcurveto{\pgfqpoint{0.643cm}{1.536cm}}{\pgfqpoint{0.608cm}{1.521cm}}{\pgfqpoint{0.582cm}{1.496cm}}
\pgfpathcurveto{\pgfqpoint{0.557cm}{1.47cm}}{\pgfqpoint{0.542cm}{1.435cm}}{\pgfqpoint{0.542cm}{1.399cm}}
\pgfpathcurveto{\pgfqpoint{0.542cm}{1.363cm}}{\pgfqpoint{0.557cm}{1.328cm}}{\pgfqpoint{0.582cm}{1.302cm}}
\pgfpathcurveto{\pgfqpoint{0.608cm}{1.276cm}}{\pgfqpoint{0.643cm}{1.262cm}}{\pgfqpoint{0.679cm}{1.262cm}}
\pgfpathcurveto{\pgfqpoint{0.715cm}{1.262cm}}{\pgfqpoint{0.75cm}{1.276cm}}{\pgfqpoint{0.775cm}{1.302cm}}
\pgfpathcurveto{\pgfqpoint{0.801cm}{1.328cm}}{\pgfqpoint{0.815cm}{1.363cm}}{\pgfqpoint{0.815cm}{1.399cm}}
\pgfusepath{fill}
\pgfpathmoveto{\pgfqpoint{1.345cm}{1.371cm}}
\pgfpathcurveto{\pgfqpoint{1.345cm}{1.408cm}}{\pgfqpoint{1.331cm}{1.442cm}}{\pgfqpoint{1.305cm}{1.468cm}}
\pgfpathcurveto{\pgfqpoint{1.28cm}{1.494cm}}{\pgfqpoint{1.245cm}{1.508cm}}{\pgfqpoint{1.209cm}{1.508cm}}
\pgfpathcurveto{\pgfqpoint{1.172cm}{1.508cm}}{\pgfqpoint{1.138cm}{1.494cm}}{\pgfqpoint{1.112cm}{1.468cm}}
\pgfpathcurveto{\pgfqpoint{1.087cm}{1.442cm}}{\pgfqpoint{1.072cm}{1.408cm}}{\pgfqpoint{1.072cm}{1.371cm}}
\pgfpathcurveto{\pgfqpoint{1.072cm}{1.335cm}}{\pgfqpoint{1.087cm}{1.3cm}}{\pgfqpoint{1.112cm}{1.274cm}}
\pgfpathcurveto{\pgfqpoint{1.138cm}{1.249cm}}{\pgfqpoint{1.172cm}{1.234cm}}{\pgfqpoint{1.209cm}{1.234cm}}
\pgfpathcurveto{\pgfqpoint{1.245cm}{1.234cm}}{\pgfqpoint{1.28cm}{1.249cm}}{\pgfqpoint{1.305cm}{1.274cm}}
\pgfpathcurveto{\pgfqpoint{1.331cm}{1.3cm}}{\pgfqpoint{1.345cm}{1.335cm}}{\pgfqpoint{1.345cm}{1.371cm}}
\pgfusepath{fill}
\begin{pgfscope}
\pgfsetdash{}{0cm}
\pgfsetlinewidth{0.818mm}
\pgfsetroundcap
\pgfsetmiterlimit{4.0}
\pgfpathmoveto{\pgfqpoint{0.682cm}{0.671cm}}
\pgfpathlineto{\pgfqpoint{0.682cm}{0.042cm}}
\pgfusepath{stroke}
\end{pgfscope}
\end{pgfscope}
\end{pgfscope}
\end{pgfscope}
\end{tikzpicture}}} + \zeta_{M, \varepsilon}) + 3\lambda b_{M,
     \varepsilon} \varphi_{M, \varepsilon}\\
     &  & + 3 X_{M, \varepsilon} (- \lambdaX_{M, \varepsilon}^{\!\resizebox{0.6em}{!}{
\begin{tikzpicture}
\pgfpathmoveto{\pgfqpoint{0cm}{-0.035cm}}
\pgfpathlineto{\pgfqpoint{1.376cm}{-0.035cm}}
\pgfpathlineto{\pgfqpoint{1.376cm}{1.552cm}}
\pgfpathlineto{\pgfqpoint{0cm}{1.552cm}}
\pgfpathclose
\pgfusepath{clip}
\begin{pgfscope}
\begin{pgfscope}
\pgfpathmoveto{\pgfqpoint{0cm}{-0.035cm}}
\pgfpathlineto{\pgfqpoint{1.376cm}{-0.035cm}}
\pgfpathlineto{\pgfqpoint{1.376cm}{1.552cm}}
\pgfpathlineto{\pgfqpoint{0cm}{1.552cm}}
\pgfpathclose
\pgfusepath{clip}
\begin{pgfscope}
\begin{pgfscope}
\pgfsetdash{}{0cm}
\pgfsetlinewidth{0.818mm}
\pgfsetroundcap
\pgfsetroundjoin
\pgfsetmiterlimit{7.0}
\definecolor{eps2pgf_color}{gray}{0}\pgfsetstrokecolor{eps2pgf_color}\pgfsetfillcolor{eps2pgf_color}
\pgfpathmoveto{\pgfqpoint{0.117cm}{1.421cm}}
\pgfpathlineto{\pgfqpoint{0.682cm}{0.671cm}}
\pgfpathlineto{\pgfqpoint{1.246cm}{1.421cm}}
\pgfusepath{stroke}
\end{pgfscope}
\definecolor{eps2pgf_color}{gray}{0}\pgfsetstrokecolor{eps2pgf_color}\pgfsetfillcolor{eps2pgf_color}
\pgfpathmoveto{\pgfqpoint{0.273cm}{1.395cm}}
\pgfpathcurveto{\pgfqpoint{0.273cm}{1.432cm}}{\pgfqpoint{0.259cm}{1.467cm}}{\pgfqpoint{0.233cm}{1.492cm}}
\pgfpathcurveto{\pgfqpoint{0.207cm}{1.518cm}}{\pgfqpoint{0.173cm}{1.532cm}}{\pgfqpoint{0.137cm}{1.532cm}}
\pgfpathcurveto{\pgfqpoint{0.1cm}{1.532cm}}{\pgfqpoint{0.066cm}{1.518cm}}{\pgfqpoint{0.04cm}{1.492cm}}
\pgfpathcurveto{\pgfqpoint{0.014cm}{1.467cm}}{\pgfqpoint{0cm}{1.432cm}}{\pgfqpoint{0cm}{1.395cm}}
\pgfpathcurveto{\pgfqpoint{0cm}{1.359cm}}{\pgfqpoint{0.014cm}{1.324cm}}{\pgfqpoint{0.04cm}{1.299cm}}
\pgfpathcurveto{\pgfqpoint{0.066cm}{1.273cm}}{\pgfqpoint{0.1cm}{1.258cm}}{\pgfqpoint{0.137cm}{1.258cm}}
\pgfpathcurveto{\pgfqpoint{0.173cm}{1.258cm}}{\pgfqpoint{0.207cm}{1.273cm}}{\pgfqpoint{0.233cm}{1.299cm}}
\pgfpathcurveto{\pgfqpoint{0.259cm}{1.324cm}}{\pgfqpoint{0.273cm}{1.359cm}}{\pgfqpoint{0.273cm}{1.395cm}}
\pgfusepath{fill}
\begin{pgfscope}
\pgfsetdash{}{0cm}
\pgfsetlinewidth{0.818mm}
\pgfsetmiterlimit{7.0}
\pgfpathmoveto{\pgfqpoint{0.682cm}{0.671cm}}
\pgfpathlineto{\pgfqpoint{0.679cm}{1.418cm}}
\pgfusepath{stroke}
\end{pgfscope}
\pgfpathmoveto{\pgfqpoint{0.815cm}{1.399cm}}
\pgfpathcurveto{\pgfqpoint{0.815cm}{1.435cm}}{\pgfqpoint{0.801cm}{1.47cm}}{\pgfqpoint{0.775cm}{1.496cm}}
\pgfpathcurveto{\pgfqpoint{0.75cm}{1.521cm}}{\pgfqpoint{0.715cm}{1.536cm}}{\pgfqpoint{0.679cm}{1.536cm}}
\pgfpathcurveto{\pgfqpoint{0.643cm}{1.536cm}}{\pgfqpoint{0.608cm}{1.521cm}}{\pgfqpoint{0.582cm}{1.496cm}}
\pgfpathcurveto{\pgfqpoint{0.557cm}{1.47cm}}{\pgfqpoint{0.542cm}{1.435cm}}{\pgfqpoint{0.542cm}{1.399cm}}
\pgfpathcurveto{\pgfqpoint{0.542cm}{1.363cm}}{\pgfqpoint{0.557cm}{1.328cm}}{\pgfqpoint{0.582cm}{1.302cm}}
\pgfpathcurveto{\pgfqpoint{0.608cm}{1.276cm}}{\pgfqpoint{0.643cm}{1.262cm}}{\pgfqpoint{0.679cm}{1.262cm}}
\pgfpathcurveto{\pgfqpoint{0.715cm}{1.262cm}}{\pgfqpoint{0.75cm}{1.276cm}}{\pgfqpoint{0.775cm}{1.302cm}}
\pgfpathcurveto{\pgfqpoint{0.801cm}{1.328cm}}{\pgfqpoint{0.815cm}{1.363cm}}{\pgfqpoint{0.815cm}{1.399cm}}
\pgfusepath{fill}
\pgfpathmoveto{\pgfqpoint{1.345cm}{1.371cm}}
\pgfpathcurveto{\pgfqpoint{1.345cm}{1.408cm}}{\pgfqpoint{1.331cm}{1.442cm}}{\pgfqpoint{1.305cm}{1.468cm}}
\pgfpathcurveto{\pgfqpoint{1.28cm}{1.494cm}}{\pgfqpoint{1.245cm}{1.508cm}}{\pgfqpoint{1.209cm}{1.508cm}}
\pgfpathcurveto{\pgfqpoint{1.172cm}{1.508cm}}{\pgfqpoint{1.138cm}{1.494cm}}{\pgfqpoint{1.112cm}{1.468cm}}
\pgfpathcurveto{\pgfqpoint{1.087cm}{1.442cm}}{\pgfqpoint{1.072cm}{1.408cm}}{\pgfqpoint{1.072cm}{1.371cm}}
\pgfpathcurveto{\pgfqpoint{1.072cm}{1.335cm}}{\pgfqpoint{1.087cm}{1.3cm}}{\pgfqpoint{1.112cm}{1.274cm}}
\pgfpathcurveto{\pgfqpoint{1.138cm}{1.249cm}}{\pgfqpoint{1.172cm}{1.234cm}}{\pgfqpoint{1.209cm}{1.234cm}}
\pgfpathcurveto{\pgfqpoint{1.245cm}{1.234cm}}{\pgfqpoint{1.28cm}{1.249cm}}{\pgfqpoint{1.305cm}{1.274cm}}
\pgfpathcurveto{\pgfqpoint{1.331cm}{1.3cm}}{\pgfqpoint{1.345cm}{1.335cm}}{\pgfqpoint{1.345cm}{1.371cm}}
\pgfusepath{fill}
\begin{pgfscope}
\pgfsetdash{}{0cm}
\pgfsetlinewidth{0.818mm}
\pgfsetroundcap
\pgfsetmiterlimit{4.0}
\pgfpathmoveto{\pgfqpoint{0.682cm}{0.671cm}}
\pgfpathlineto{\pgfqpoint{0.682cm}{0.042cm}}
\pgfusepath{stroke}
\end{pgfscope}
\end{pgfscope}
\end{pgfscope}
\end{pgfscope}
\end{tikzpicture}}} +
     \zeta_{M, \varepsilon})^2 + (-\lambda X_{M, \varepsilon}^{\!\resizebox{0.6em}{!}{
\begin{tikzpicture}
\pgfpathmoveto{\pgfqpoint{0cm}{-0.035cm}}
\pgfpathlineto{\pgfqpoint{1.376cm}{-0.035cm}}
\pgfpathlineto{\pgfqpoint{1.376cm}{1.552cm}}
\pgfpathlineto{\pgfqpoint{0cm}{1.552cm}}
\pgfpathclose
\pgfusepath{clip}
\begin{pgfscope}
\begin{pgfscope}
\pgfpathmoveto{\pgfqpoint{0cm}{-0.035cm}}
\pgfpathlineto{\pgfqpoint{1.376cm}{-0.035cm}}
\pgfpathlineto{\pgfqpoint{1.376cm}{1.552cm}}
\pgfpathlineto{\pgfqpoint{0cm}{1.552cm}}
\pgfpathclose
\pgfusepath{clip}
\begin{pgfscope}
\begin{pgfscope}
\pgfsetdash{}{0cm}
\pgfsetlinewidth{0.818mm}
\pgfsetroundcap
\pgfsetroundjoin
\pgfsetmiterlimit{7.0}
\definecolor{eps2pgf_color}{gray}{0}\pgfsetstrokecolor{eps2pgf_color}\pgfsetfillcolor{eps2pgf_color}
\pgfpathmoveto{\pgfqpoint{0.117cm}{1.421cm}}
\pgfpathlineto{\pgfqpoint{0.682cm}{0.671cm}}
\pgfpathlineto{\pgfqpoint{1.246cm}{1.421cm}}
\pgfusepath{stroke}
\end{pgfscope}
\definecolor{eps2pgf_color}{gray}{0}\pgfsetstrokecolor{eps2pgf_color}\pgfsetfillcolor{eps2pgf_color}
\pgfpathmoveto{\pgfqpoint{0.273cm}{1.395cm}}
\pgfpathcurveto{\pgfqpoint{0.273cm}{1.432cm}}{\pgfqpoint{0.259cm}{1.467cm}}{\pgfqpoint{0.233cm}{1.492cm}}
\pgfpathcurveto{\pgfqpoint{0.207cm}{1.518cm}}{\pgfqpoint{0.173cm}{1.532cm}}{\pgfqpoint{0.137cm}{1.532cm}}
\pgfpathcurveto{\pgfqpoint{0.1cm}{1.532cm}}{\pgfqpoint{0.066cm}{1.518cm}}{\pgfqpoint{0.04cm}{1.492cm}}
\pgfpathcurveto{\pgfqpoint{0.014cm}{1.467cm}}{\pgfqpoint{0cm}{1.432cm}}{\pgfqpoint{0cm}{1.395cm}}
\pgfpathcurveto{\pgfqpoint{0cm}{1.359cm}}{\pgfqpoint{0.014cm}{1.324cm}}{\pgfqpoint{0.04cm}{1.299cm}}
\pgfpathcurveto{\pgfqpoint{0.066cm}{1.273cm}}{\pgfqpoint{0.1cm}{1.258cm}}{\pgfqpoint{0.137cm}{1.258cm}}
\pgfpathcurveto{\pgfqpoint{0.173cm}{1.258cm}}{\pgfqpoint{0.207cm}{1.273cm}}{\pgfqpoint{0.233cm}{1.299cm}}
\pgfpathcurveto{\pgfqpoint{0.259cm}{1.324cm}}{\pgfqpoint{0.273cm}{1.359cm}}{\pgfqpoint{0.273cm}{1.395cm}}
\pgfusepath{fill}
\begin{pgfscope}
\pgfsetdash{}{0cm}
\pgfsetlinewidth{0.818mm}
\pgfsetmiterlimit{7.0}
\pgfpathmoveto{\pgfqpoint{0.682cm}{0.671cm}}
\pgfpathlineto{\pgfqpoint{0.679cm}{1.418cm}}
\pgfusepath{stroke}
\end{pgfscope}
\pgfpathmoveto{\pgfqpoint{0.815cm}{1.399cm}}
\pgfpathcurveto{\pgfqpoint{0.815cm}{1.435cm}}{\pgfqpoint{0.801cm}{1.47cm}}{\pgfqpoint{0.775cm}{1.496cm}}
\pgfpathcurveto{\pgfqpoint{0.75cm}{1.521cm}}{\pgfqpoint{0.715cm}{1.536cm}}{\pgfqpoint{0.679cm}{1.536cm}}
\pgfpathcurveto{\pgfqpoint{0.643cm}{1.536cm}}{\pgfqpoint{0.608cm}{1.521cm}}{\pgfqpoint{0.582cm}{1.496cm}}
\pgfpathcurveto{\pgfqpoint{0.557cm}{1.47cm}}{\pgfqpoint{0.542cm}{1.435cm}}{\pgfqpoint{0.542cm}{1.399cm}}
\pgfpathcurveto{\pgfqpoint{0.542cm}{1.363cm}}{\pgfqpoint{0.557cm}{1.328cm}}{\pgfqpoint{0.582cm}{1.302cm}}
\pgfpathcurveto{\pgfqpoint{0.608cm}{1.276cm}}{\pgfqpoint{0.643cm}{1.262cm}}{\pgfqpoint{0.679cm}{1.262cm}}
\pgfpathcurveto{\pgfqpoint{0.715cm}{1.262cm}}{\pgfqpoint{0.75cm}{1.276cm}}{\pgfqpoint{0.775cm}{1.302cm}}
\pgfpathcurveto{\pgfqpoint{0.801cm}{1.328cm}}{\pgfqpoint{0.815cm}{1.363cm}}{\pgfqpoint{0.815cm}{1.399cm}}
\pgfusepath{fill}
\pgfpathmoveto{\pgfqpoint{1.345cm}{1.371cm}}
\pgfpathcurveto{\pgfqpoint{1.345cm}{1.408cm}}{\pgfqpoint{1.331cm}{1.442cm}}{\pgfqpoint{1.305cm}{1.468cm}}
\pgfpathcurveto{\pgfqpoint{1.28cm}{1.494cm}}{\pgfqpoint{1.245cm}{1.508cm}}{\pgfqpoint{1.209cm}{1.508cm}}
\pgfpathcurveto{\pgfqpoint{1.172cm}{1.508cm}}{\pgfqpoint{1.138cm}{1.494cm}}{\pgfqpoint{1.112cm}{1.468cm}}
\pgfpathcurveto{\pgfqpoint{1.087cm}{1.442cm}}{\pgfqpoint{1.072cm}{1.408cm}}{\pgfqpoint{1.072cm}{1.371cm}}
\pgfpathcurveto{\pgfqpoint{1.072cm}{1.335cm}}{\pgfqpoint{1.087cm}{1.3cm}}{\pgfqpoint{1.112cm}{1.274cm}}
\pgfpathcurveto{\pgfqpoint{1.138cm}{1.249cm}}{\pgfqpoint{1.172cm}{1.234cm}}{\pgfqpoint{1.209cm}{1.234cm}}
\pgfpathcurveto{\pgfqpoint{1.245cm}{1.234cm}}{\pgfqpoint{1.28cm}{1.249cm}}{\pgfqpoint{1.305cm}{1.274cm}}
\pgfpathcurveto{\pgfqpoint{1.331cm}{1.3cm}}{\pgfqpoint{1.345cm}{1.335cm}}{\pgfqpoint{1.345cm}{1.371cm}}
\pgfusepath{fill}
\begin{pgfscope}
\pgfsetdash{}{0cm}
\pgfsetlinewidth{0.818mm}
\pgfsetroundcap
\pgfsetmiterlimit{4.0}
\pgfpathmoveto{\pgfqpoint{0.682cm}{0.671cm}}
\pgfpathlineto{\pgfqpoint{0.682cm}{0.042cm}}
\pgfusepath{stroke}
\end{pgfscope}
\end{pgfscope}
\end{pgfscope}
\end{pgfscope}
\end{tikzpicture}}} + \zeta_{M,
     \varepsilon})^3 .
   \end{array} \]
Next, we use the paraproducts and paracontrolled ansatz to control the various
resonant products. For the renormalized resonant product $3 \llbracket X_{M,
\varepsilon}^2 \rrbracket \circ (-\lambda X_{M, \varepsilon}^{\!\resizebox{0.6em}{!}{
\begin{tikzpicture}
\pgfpathmoveto{\pgfqpoint{0cm}{-0.035cm}}
\pgfpathlineto{\pgfqpoint{1.376cm}{-0.035cm}}
\pgfpathlineto{\pgfqpoint{1.376cm}{1.552cm}}
\pgfpathlineto{\pgfqpoint{0cm}{1.552cm}}
\pgfpathclose
\pgfusepath{clip}
\begin{pgfscope}
\begin{pgfscope}
\pgfpathmoveto{\pgfqpoint{0cm}{-0.035cm}}
\pgfpathlineto{\pgfqpoint{1.376cm}{-0.035cm}}
\pgfpathlineto{\pgfqpoint{1.376cm}{1.552cm}}
\pgfpathlineto{\pgfqpoint{0cm}{1.552cm}}
\pgfpathclose
\pgfusepath{clip}
\begin{pgfscope}
\begin{pgfscope}
\pgfsetdash{}{0cm}
\pgfsetlinewidth{0.818mm}
\pgfsetroundcap
\pgfsetroundjoin
\pgfsetmiterlimit{7.0}
\definecolor{eps2pgf_color}{gray}{0}\pgfsetstrokecolor{eps2pgf_color}\pgfsetfillcolor{eps2pgf_color}
\pgfpathmoveto{\pgfqpoint{0.117cm}{1.421cm}}
\pgfpathlineto{\pgfqpoint{0.682cm}{0.671cm}}
\pgfpathlineto{\pgfqpoint{1.246cm}{1.421cm}}
\pgfusepath{stroke}
\end{pgfscope}
\definecolor{eps2pgf_color}{gray}{0}\pgfsetstrokecolor{eps2pgf_color}\pgfsetfillcolor{eps2pgf_color}
\pgfpathmoveto{\pgfqpoint{0.273cm}{1.395cm}}
\pgfpathcurveto{\pgfqpoint{0.273cm}{1.432cm}}{\pgfqpoint{0.259cm}{1.467cm}}{\pgfqpoint{0.233cm}{1.492cm}}
\pgfpathcurveto{\pgfqpoint{0.207cm}{1.518cm}}{\pgfqpoint{0.173cm}{1.532cm}}{\pgfqpoint{0.137cm}{1.532cm}}
\pgfpathcurveto{\pgfqpoint{0.1cm}{1.532cm}}{\pgfqpoint{0.066cm}{1.518cm}}{\pgfqpoint{0.04cm}{1.492cm}}
\pgfpathcurveto{\pgfqpoint{0.014cm}{1.467cm}}{\pgfqpoint{0cm}{1.432cm}}{\pgfqpoint{0cm}{1.395cm}}
\pgfpathcurveto{\pgfqpoint{0cm}{1.359cm}}{\pgfqpoint{0.014cm}{1.324cm}}{\pgfqpoint{0.04cm}{1.299cm}}
\pgfpathcurveto{\pgfqpoint{0.066cm}{1.273cm}}{\pgfqpoint{0.1cm}{1.258cm}}{\pgfqpoint{0.137cm}{1.258cm}}
\pgfpathcurveto{\pgfqpoint{0.173cm}{1.258cm}}{\pgfqpoint{0.207cm}{1.273cm}}{\pgfqpoint{0.233cm}{1.299cm}}
\pgfpathcurveto{\pgfqpoint{0.259cm}{1.324cm}}{\pgfqpoint{0.273cm}{1.359cm}}{\pgfqpoint{0.273cm}{1.395cm}}
\pgfusepath{fill}
\begin{pgfscope}
\pgfsetdash{}{0cm}
\pgfsetlinewidth{0.818mm}
\pgfsetmiterlimit{7.0}
\pgfpathmoveto{\pgfqpoint{0.682cm}{0.671cm}}
\pgfpathlineto{\pgfqpoint{0.679cm}{1.418cm}}
\pgfusepath{stroke}
\end{pgfscope}
\pgfpathmoveto{\pgfqpoint{0.815cm}{1.399cm}}
\pgfpathcurveto{\pgfqpoint{0.815cm}{1.435cm}}{\pgfqpoint{0.801cm}{1.47cm}}{\pgfqpoint{0.775cm}{1.496cm}}
\pgfpathcurveto{\pgfqpoint{0.75cm}{1.521cm}}{\pgfqpoint{0.715cm}{1.536cm}}{\pgfqpoint{0.679cm}{1.536cm}}
\pgfpathcurveto{\pgfqpoint{0.643cm}{1.536cm}}{\pgfqpoint{0.608cm}{1.521cm}}{\pgfqpoint{0.582cm}{1.496cm}}
\pgfpathcurveto{\pgfqpoint{0.557cm}{1.47cm}}{\pgfqpoint{0.542cm}{1.435cm}}{\pgfqpoint{0.542cm}{1.399cm}}
\pgfpathcurveto{\pgfqpoint{0.542cm}{1.363cm}}{\pgfqpoint{0.557cm}{1.328cm}}{\pgfqpoint{0.582cm}{1.302cm}}
\pgfpathcurveto{\pgfqpoint{0.608cm}{1.276cm}}{\pgfqpoint{0.643cm}{1.262cm}}{\pgfqpoint{0.679cm}{1.262cm}}
\pgfpathcurveto{\pgfqpoint{0.715cm}{1.262cm}}{\pgfqpoint{0.75cm}{1.276cm}}{\pgfqpoint{0.775cm}{1.302cm}}
\pgfpathcurveto{\pgfqpoint{0.801cm}{1.328cm}}{\pgfqpoint{0.815cm}{1.363cm}}{\pgfqpoint{0.815cm}{1.399cm}}
\pgfusepath{fill}
\pgfpathmoveto{\pgfqpoint{1.345cm}{1.371cm}}
\pgfpathcurveto{\pgfqpoint{1.345cm}{1.408cm}}{\pgfqpoint{1.331cm}{1.442cm}}{\pgfqpoint{1.305cm}{1.468cm}}
\pgfpathcurveto{\pgfqpoint{1.28cm}{1.494cm}}{\pgfqpoint{1.245cm}{1.508cm}}{\pgfqpoint{1.209cm}{1.508cm}}
\pgfpathcurveto{\pgfqpoint{1.172cm}{1.508cm}}{\pgfqpoint{1.138cm}{1.494cm}}{\pgfqpoint{1.112cm}{1.468cm}}
\pgfpathcurveto{\pgfqpoint{1.087cm}{1.442cm}}{\pgfqpoint{1.072cm}{1.408cm}}{\pgfqpoint{1.072cm}{1.371cm}}
\pgfpathcurveto{\pgfqpoint{1.072cm}{1.335cm}}{\pgfqpoint{1.087cm}{1.3cm}}{\pgfqpoint{1.112cm}{1.274cm}}
\pgfpathcurveto{\pgfqpoint{1.138cm}{1.249cm}}{\pgfqpoint{1.172cm}{1.234cm}}{\pgfqpoint{1.209cm}{1.234cm}}
\pgfpathcurveto{\pgfqpoint{1.245cm}{1.234cm}}{\pgfqpoint{1.28cm}{1.249cm}}{\pgfqpoint{1.305cm}{1.274cm}}
\pgfpathcurveto{\pgfqpoint{1.331cm}{1.3cm}}{\pgfqpoint{1.345cm}{1.335cm}}{\pgfqpoint{1.345cm}{1.371cm}}
\pgfusepath{fill}
\begin{pgfscope}
\pgfsetdash{}{0cm}
\pgfsetlinewidth{0.818mm}
\pgfsetroundcap
\pgfsetmiterlimit{4.0}
\pgfpathmoveto{\pgfqpoint{0.682cm}{0.671cm}}
\pgfpathlineto{\pgfqpoint{0.682cm}{0.042cm}}
\pgfusepath{stroke}
\end{pgfscope}
\end{pgfscope}
\end{pgfscope}
\end{pgfscope}
\end{tikzpicture}}} + \zeta_{M,
\varepsilon}) + 3\lambda b_{M, \varepsilon} \varphi_{M, \varepsilon}$ we first recall
that
\[ \varphi_{M, \varepsilon} = X_{M, \varepsilon} + Y_{M, \varepsilon} +
   \phi_{M, \varepsilon}, \qquad \phi_{M, \varepsilon} = - 3\lambda X_{M,
   \varepsilon}^{\!\resizebox{0.6em}{!}{
\begin{tikzpicture}
\pgfpathmoveto{\pgfqpoint{0cm}{0cm}}
\pgfpathlineto{\pgfqpoint{1.376cm}{0cm}}
\pgfpathlineto{\pgfqpoint{1.376cm}{1.588cm}}
\pgfpathlineto{\pgfqpoint{0cm}{1.588cm}}
\pgfpathclose
\pgfusepath{clip}
\begin{pgfscope}
\begin{pgfscope}
\pgfpathmoveto{\pgfqpoint{0cm}{0cm}}
\pgfpathlineto{\pgfqpoint{1.376cm}{0cm}}
\pgfpathlineto{\pgfqpoint{1.376cm}{1.588cm}}
\pgfpathlineto{\pgfqpoint{0cm}{1.588cm}}
\pgfpathclose
\pgfusepath{clip}
\begin{pgfscope}
\begin{pgfscope}
\definecolor{eps2pgf_color}{gray}{0.976471}\pgfsetstrokecolor{eps2pgf_color}\pgfsetfillcolor{eps2pgf_color}
\pgfpathmoveto{\pgfqpoint{0cm}{0cm}}
\pgfpathlineto{\pgfqpoint{1.376cm}{0cm}}
\pgfpathlineto{\pgfqpoint{1.376cm}{1.588cm}}
\pgfpathlineto{\pgfqpoint{0cm}{1.588cm}}
\pgfpathclose
\pgfusepath{fill}
\end{pgfscope}
\begin{pgfscope}
\pgfsetdash{}{0cm}
\pgfsetlinewidth{0.818mm}
\pgfsetroundcap
\pgfsetroundjoin
\pgfsetmiterlimit{7.0}
\definecolor{eps2pgf_color}{gray}{0}\pgfsetstrokecolor{eps2pgf_color}\pgfsetfillcolor{eps2pgf_color}
\pgfpathmoveto{\pgfqpoint{0.117cm}{1.476cm}}
\pgfpathlineto{\pgfqpoint{0.682cm}{0.726cm}}
\pgfpathlineto{\pgfqpoint{1.246cm}{1.476cm}}
\pgfusepath{stroke}
\end{pgfscope}
\definecolor{eps2pgf_color}{gray}{0}\pgfsetstrokecolor{eps2pgf_color}\pgfsetfillcolor{eps2pgf_color}
\pgfpathmoveto{\pgfqpoint{0.273cm}{1.451cm}}
\pgfpathcurveto{\pgfqpoint{0.273cm}{1.487cm}}{\pgfqpoint{0.259cm}{1.522cm}}{\pgfqpoint{0.233cm}{1.547cm}}
\pgfpathcurveto{\pgfqpoint{0.207cm}{1.573cm}}{\pgfqpoint{0.173cm}{1.588cm}}{\pgfqpoint{0.137cm}{1.588cm}}
\pgfpathcurveto{\pgfqpoint{0.1cm}{1.588cm}}{\pgfqpoint{0.066cm}{1.573cm}}{\pgfqpoint{0.04cm}{1.547cm}}
\pgfpathcurveto{\pgfqpoint{0.014cm}{1.522cm}}{\pgfqpoint{0cm}{1.487cm}}{\pgfqpoint{0cm}{1.451cm}}
\pgfpathcurveto{\pgfqpoint{0cm}{1.414cm}}{\pgfqpoint{0.014cm}{1.379cm}}{\pgfqpoint{0.04cm}{1.354cm}}
\pgfpathcurveto{\pgfqpoint{0.066cm}{1.328cm}}{\pgfqpoint{0.1cm}{1.314cm}}{\pgfqpoint{0.137cm}{1.314cm}}
\pgfpathcurveto{\pgfqpoint{0.173cm}{1.314cm}}{\pgfqpoint{0.207cm}{1.328cm}}{\pgfqpoint{0.233cm}{1.354cm}}
\pgfpathcurveto{\pgfqpoint{0.259cm}{1.379cm}}{\pgfqpoint{0.273cm}{1.414cm}}{\pgfqpoint{0.273cm}{1.451cm}}
\pgfusepath{fill}
\pgfpathmoveto{\pgfqpoint{1.345cm}{1.426cm}}
\pgfpathcurveto{\pgfqpoint{1.345cm}{1.463cm}}{\pgfqpoint{1.331cm}{1.497cm}}{\pgfqpoint{1.305cm}{1.523cm}}
\pgfpathcurveto{\pgfqpoint{1.28cm}{1.549cm}}{\pgfqpoint{1.245cm}{1.563cm}}{\pgfqpoint{1.209cm}{1.563cm}}
\pgfpathcurveto{\pgfqpoint{1.172cm}{1.563cm}}{\pgfqpoint{1.138cm}{1.549cm}}{\pgfqpoint{1.112cm}{1.523cm}}
\pgfpathcurveto{\pgfqpoint{1.087cm}{1.497cm}}{\pgfqpoint{1.072cm}{1.463cm}}{\pgfqpoint{1.072cm}{1.426cm}}
\pgfpathcurveto{\pgfqpoint{1.072cm}{1.39cm}}{\pgfqpoint{1.087cm}{1.355cm}}{\pgfqpoint{1.112cm}{1.329cm}}
\pgfpathcurveto{\pgfqpoint{1.138cm}{1.304cm}}{\pgfqpoint{1.172cm}{1.289cm}}{\pgfqpoint{1.209cm}{1.289cm}}
\pgfpathcurveto{\pgfqpoint{1.245cm}{1.289cm}}{\pgfqpoint{1.28cm}{1.304cm}}{\pgfqpoint{1.305cm}{1.329cm}}
\pgfpathcurveto{\pgfqpoint{1.331cm}{1.355cm}}{\pgfqpoint{1.345cm}{1.39cm}}{\pgfqpoint{1.345cm}{1.426cm}}
\pgfusepath{fill}
\begin{pgfscope}
\pgfsetdash{}{0cm}
\pgfsetlinewidth{0.818mm}
\pgfsetroundcap
\pgfsetmiterlimit{4.0}
\pgfpathmoveto{\pgfqpoint{0.682cm}{0.726cm}}
\pgfpathlineto{\pgfqpoint{0.682cm}{0.097cm}}
\pgfusepath{stroke}
\end{pgfscope}
\end{pgfscope}
\end{pgfscope}
\end{pgfscope}
\end{tikzpicture}}} \succ \phi_{M, \varepsilon} + \chi_{M, \varepsilon} . \]
Therefore using the definition of $Z_{M, \varepsilon}$ in~{\eqref{eq:def-Z}}
we have
\[ \begin{array}{lll}
     3 \llbracket X_{M, \varepsilon}^2 \rrbracket \circ (-\lambda X_{M,
     \varepsilon}^{\!\resizebox{0.6em}{!}{
\begin{tikzpicture}
\pgfpathmoveto{\pgfqpoint{0cm}{-0.035cm}}
\pgfpathlineto{\pgfqpoint{1.376cm}{-0.035cm}}
\pgfpathlineto{\pgfqpoint{1.376cm}{1.552cm}}
\pgfpathlineto{\pgfqpoint{0cm}{1.552cm}}
\pgfpathclose
\pgfusepath{clip}
\begin{pgfscope}
\begin{pgfscope}
\pgfpathmoveto{\pgfqpoint{0cm}{-0.035cm}}
\pgfpathlineto{\pgfqpoint{1.376cm}{-0.035cm}}
\pgfpathlineto{\pgfqpoint{1.376cm}{1.552cm}}
\pgfpathlineto{\pgfqpoint{0cm}{1.552cm}}
\pgfpathclose
\pgfusepath{clip}
\begin{pgfscope}
\begin{pgfscope}
\pgfsetdash{}{0cm}
\pgfsetlinewidth{0.818mm}
\pgfsetroundcap
\pgfsetroundjoin
\pgfsetmiterlimit{7.0}
\definecolor{eps2pgf_color}{gray}{0}\pgfsetstrokecolor{eps2pgf_color}\pgfsetfillcolor{eps2pgf_color}
\pgfpathmoveto{\pgfqpoint{0.117cm}{1.421cm}}
\pgfpathlineto{\pgfqpoint{0.682cm}{0.671cm}}
\pgfpathlineto{\pgfqpoint{1.246cm}{1.421cm}}
\pgfusepath{stroke}
\end{pgfscope}
\definecolor{eps2pgf_color}{gray}{0}\pgfsetstrokecolor{eps2pgf_color}\pgfsetfillcolor{eps2pgf_color}
\pgfpathmoveto{\pgfqpoint{0.273cm}{1.395cm}}
\pgfpathcurveto{\pgfqpoint{0.273cm}{1.432cm}}{\pgfqpoint{0.259cm}{1.467cm}}{\pgfqpoint{0.233cm}{1.492cm}}
\pgfpathcurveto{\pgfqpoint{0.207cm}{1.518cm}}{\pgfqpoint{0.173cm}{1.532cm}}{\pgfqpoint{0.137cm}{1.532cm}}
\pgfpathcurveto{\pgfqpoint{0.1cm}{1.532cm}}{\pgfqpoint{0.066cm}{1.518cm}}{\pgfqpoint{0.04cm}{1.492cm}}
\pgfpathcurveto{\pgfqpoint{0.014cm}{1.467cm}}{\pgfqpoint{0cm}{1.432cm}}{\pgfqpoint{0cm}{1.395cm}}
\pgfpathcurveto{\pgfqpoint{0cm}{1.359cm}}{\pgfqpoint{0.014cm}{1.324cm}}{\pgfqpoint{0.04cm}{1.299cm}}
\pgfpathcurveto{\pgfqpoint{0.066cm}{1.273cm}}{\pgfqpoint{0.1cm}{1.258cm}}{\pgfqpoint{0.137cm}{1.258cm}}
\pgfpathcurveto{\pgfqpoint{0.173cm}{1.258cm}}{\pgfqpoint{0.207cm}{1.273cm}}{\pgfqpoint{0.233cm}{1.299cm}}
\pgfpathcurveto{\pgfqpoint{0.259cm}{1.324cm}}{\pgfqpoint{0.273cm}{1.359cm}}{\pgfqpoint{0.273cm}{1.395cm}}
\pgfusepath{fill}
\begin{pgfscope}
\pgfsetdash{}{0cm}
\pgfsetlinewidth{0.818mm}
\pgfsetmiterlimit{7.0}
\pgfpathmoveto{\pgfqpoint{0.682cm}{0.671cm}}
\pgfpathlineto{\pgfqpoint{0.679cm}{1.418cm}}
\pgfusepath{stroke}
\end{pgfscope}
\pgfpathmoveto{\pgfqpoint{0.815cm}{1.399cm}}
\pgfpathcurveto{\pgfqpoint{0.815cm}{1.435cm}}{\pgfqpoint{0.801cm}{1.47cm}}{\pgfqpoint{0.775cm}{1.496cm}}
\pgfpathcurveto{\pgfqpoint{0.75cm}{1.521cm}}{\pgfqpoint{0.715cm}{1.536cm}}{\pgfqpoint{0.679cm}{1.536cm}}
\pgfpathcurveto{\pgfqpoint{0.643cm}{1.536cm}}{\pgfqpoint{0.608cm}{1.521cm}}{\pgfqpoint{0.582cm}{1.496cm}}
\pgfpathcurveto{\pgfqpoint{0.557cm}{1.47cm}}{\pgfqpoint{0.542cm}{1.435cm}}{\pgfqpoint{0.542cm}{1.399cm}}
\pgfpathcurveto{\pgfqpoint{0.542cm}{1.363cm}}{\pgfqpoint{0.557cm}{1.328cm}}{\pgfqpoint{0.582cm}{1.302cm}}
\pgfpathcurveto{\pgfqpoint{0.608cm}{1.276cm}}{\pgfqpoint{0.643cm}{1.262cm}}{\pgfqpoint{0.679cm}{1.262cm}}
\pgfpathcurveto{\pgfqpoint{0.715cm}{1.262cm}}{\pgfqpoint{0.75cm}{1.276cm}}{\pgfqpoint{0.775cm}{1.302cm}}
\pgfpathcurveto{\pgfqpoint{0.801cm}{1.328cm}}{\pgfqpoint{0.815cm}{1.363cm}}{\pgfqpoint{0.815cm}{1.399cm}}
\pgfusepath{fill}
\pgfpathmoveto{\pgfqpoint{1.345cm}{1.371cm}}
\pgfpathcurveto{\pgfqpoint{1.345cm}{1.408cm}}{\pgfqpoint{1.331cm}{1.442cm}}{\pgfqpoint{1.305cm}{1.468cm}}
\pgfpathcurveto{\pgfqpoint{1.28cm}{1.494cm}}{\pgfqpoint{1.245cm}{1.508cm}}{\pgfqpoint{1.209cm}{1.508cm}}
\pgfpathcurveto{\pgfqpoint{1.172cm}{1.508cm}}{\pgfqpoint{1.138cm}{1.494cm}}{\pgfqpoint{1.112cm}{1.468cm}}
\pgfpathcurveto{\pgfqpoint{1.087cm}{1.442cm}}{\pgfqpoint{1.072cm}{1.408cm}}{\pgfqpoint{1.072cm}{1.371cm}}
\pgfpathcurveto{\pgfqpoint{1.072cm}{1.335cm}}{\pgfqpoint{1.087cm}{1.3cm}}{\pgfqpoint{1.112cm}{1.274cm}}
\pgfpathcurveto{\pgfqpoint{1.138cm}{1.249cm}}{\pgfqpoint{1.172cm}{1.234cm}}{\pgfqpoint{1.209cm}{1.234cm}}
\pgfpathcurveto{\pgfqpoint{1.245cm}{1.234cm}}{\pgfqpoint{1.28cm}{1.249cm}}{\pgfqpoint{1.305cm}{1.274cm}}
\pgfpathcurveto{\pgfqpoint{1.331cm}{1.3cm}}{\pgfqpoint{1.345cm}{1.335cm}}{\pgfqpoint{1.345cm}{1.371cm}}
\pgfusepath{fill}
\begin{pgfscope}
\pgfsetdash{}{0cm}
\pgfsetlinewidth{0.818mm}
\pgfsetroundcap
\pgfsetmiterlimit{4.0}
\pgfpathmoveto{\pgfqpoint{0.682cm}{0.671cm}}
\pgfpathlineto{\pgfqpoint{0.682cm}{0.042cm}}
\pgfusepath{stroke}
\end{pgfscope}
\end{pgfscope}
\end{pgfscope}
\end{pgfscope}
\end{tikzpicture}}} + \zeta_{M, \varepsilon}) + 3\lambda b_{M, \varepsilon}
     \varphi_{M, \varepsilon} & = & 3 \llbracket X_{M, \varepsilon}^2
     \rrbracket \circ (Y_{M, \varepsilon} + \phi_{M, \varepsilon}) + 3\lambda b_{M,
     \varepsilon} \varphi_{M, \varepsilon}\\
     & = & \underbrace{3 \llbracket X_{M, \varepsilon}^2 \rrbracket \circ
     Y_{M, \varepsilon} + 3\lambda b_{M, \varepsilon} (X_{M, \varepsilon} + Y_{M,
     \varepsilon})}_{- \lambda Z_{M, \varepsilon}}\\
     &  & + 3 \llbracket X_{M, \varepsilon}^2 \rrbracket \circ \phi_{M,
     \varepsilon} + 3\lambda b_{M, \varepsilon} \phi_{M, \varepsilon}
   \end{array} \]
and
\[ 3 \llbracket X_{M, \varepsilon}^2 \rrbracket \circ \phi_{M, \varepsilon} +
   3\lambda b_{M, \varepsilon} \phi_{M, \varepsilon} = 3 \llbracket X_{M,
   \varepsilon}^2 \rrbracket \circ ( - 3\lambda X_{M, \varepsilon}^{\!\resizebox{0.6em}{!}{
\begin{tikzpicture}
\pgfpathmoveto{\pgfqpoint{0cm}{0cm}}
\pgfpathlineto{\pgfqpoint{1.376cm}{0cm}}
\pgfpathlineto{\pgfqpoint{1.376cm}{1.588cm}}
\pgfpathlineto{\pgfqpoint{0cm}{1.588cm}}
\pgfpathclose
\pgfusepath{clip}
\begin{pgfscope}
\begin{pgfscope}
\pgfpathmoveto{\pgfqpoint{0cm}{0cm}}
\pgfpathlineto{\pgfqpoint{1.376cm}{0cm}}
\pgfpathlineto{\pgfqpoint{1.376cm}{1.588cm}}
\pgfpathlineto{\pgfqpoint{0cm}{1.588cm}}
\pgfpathclose
\pgfusepath{clip}
\begin{pgfscope}
\begin{pgfscope}
\definecolor{eps2pgf_color}{gray}{0.976471}\pgfsetstrokecolor{eps2pgf_color}\pgfsetfillcolor{eps2pgf_color}
\pgfpathmoveto{\pgfqpoint{0cm}{0cm}}
\pgfpathlineto{\pgfqpoint{1.376cm}{0cm}}
\pgfpathlineto{\pgfqpoint{1.376cm}{1.588cm}}
\pgfpathlineto{\pgfqpoint{0cm}{1.588cm}}
\pgfpathclose
\pgfusepath{fill}
\end{pgfscope}
\begin{pgfscope}
\pgfsetdash{}{0cm}
\pgfsetlinewidth{0.818mm}
\pgfsetroundcap
\pgfsetroundjoin
\pgfsetmiterlimit{7.0}
\definecolor{eps2pgf_color}{gray}{0}\pgfsetstrokecolor{eps2pgf_color}\pgfsetfillcolor{eps2pgf_color}
\pgfpathmoveto{\pgfqpoint{0.117cm}{1.476cm}}
\pgfpathlineto{\pgfqpoint{0.682cm}{0.726cm}}
\pgfpathlineto{\pgfqpoint{1.246cm}{1.476cm}}
\pgfusepath{stroke}
\end{pgfscope}
\definecolor{eps2pgf_color}{gray}{0}\pgfsetstrokecolor{eps2pgf_color}\pgfsetfillcolor{eps2pgf_color}
\pgfpathmoveto{\pgfqpoint{0.273cm}{1.451cm}}
\pgfpathcurveto{\pgfqpoint{0.273cm}{1.487cm}}{\pgfqpoint{0.259cm}{1.522cm}}{\pgfqpoint{0.233cm}{1.547cm}}
\pgfpathcurveto{\pgfqpoint{0.207cm}{1.573cm}}{\pgfqpoint{0.173cm}{1.588cm}}{\pgfqpoint{0.137cm}{1.588cm}}
\pgfpathcurveto{\pgfqpoint{0.1cm}{1.588cm}}{\pgfqpoint{0.066cm}{1.573cm}}{\pgfqpoint{0.04cm}{1.547cm}}
\pgfpathcurveto{\pgfqpoint{0.014cm}{1.522cm}}{\pgfqpoint{0cm}{1.487cm}}{\pgfqpoint{0cm}{1.451cm}}
\pgfpathcurveto{\pgfqpoint{0cm}{1.414cm}}{\pgfqpoint{0.014cm}{1.379cm}}{\pgfqpoint{0.04cm}{1.354cm}}
\pgfpathcurveto{\pgfqpoint{0.066cm}{1.328cm}}{\pgfqpoint{0.1cm}{1.314cm}}{\pgfqpoint{0.137cm}{1.314cm}}
\pgfpathcurveto{\pgfqpoint{0.173cm}{1.314cm}}{\pgfqpoint{0.207cm}{1.328cm}}{\pgfqpoint{0.233cm}{1.354cm}}
\pgfpathcurveto{\pgfqpoint{0.259cm}{1.379cm}}{\pgfqpoint{0.273cm}{1.414cm}}{\pgfqpoint{0.273cm}{1.451cm}}
\pgfusepath{fill}
\pgfpathmoveto{\pgfqpoint{1.345cm}{1.426cm}}
\pgfpathcurveto{\pgfqpoint{1.345cm}{1.463cm}}{\pgfqpoint{1.331cm}{1.497cm}}{\pgfqpoint{1.305cm}{1.523cm}}
\pgfpathcurveto{\pgfqpoint{1.28cm}{1.549cm}}{\pgfqpoint{1.245cm}{1.563cm}}{\pgfqpoint{1.209cm}{1.563cm}}
\pgfpathcurveto{\pgfqpoint{1.172cm}{1.563cm}}{\pgfqpoint{1.138cm}{1.549cm}}{\pgfqpoint{1.112cm}{1.523cm}}
\pgfpathcurveto{\pgfqpoint{1.087cm}{1.497cm}}{\pgfqpoint{1.072cm}{1.463cm}}{\pgfqpoint{1.072cm}{1.426cm}}
\pgfpathcurveto{\pgfqpoint{1.072cm}{1.39cm}}{\pgfqpoint{1.087cm}{1.355cm}}{\pgfqpoint{1.112cm}{1.329cm}}
\pgfpathcurveto{\pgfqpoint{1.138cm}{1.304cm}}{\pgfqpoint{1.172cm}{1.289cm}}{\pgfqpoint{1.209cm}{1.289cm}}
\pgfpathcurveto{\pgfqpoint{1.245cm}{1.289cm}}{\pgfqpoint{1.28cm}{1.304cm}}{\pgfqpoint{1.305cm}{1.329cm}}
\pgfpathcurveto{\pgfqpoint{1.331cm}{1.355cm}}{\pgfqpoint{1.345cm}{1.39cm}}{\pgfqpoint{1.345cm}{1.426cm}}
\pgfusepath{fill}
\begin{pgfscope}
\pgfsetdash{}{0cm}
\pgfsetlinewidth{0.818mm}
\pgfsetroundcap
\pgfsetmiterlimit{4.0}
\pgfpathmoveto{\pgfqpoint{0.682cm}{0.726cm}}
\pgfpathlineto{\pgfqpoint{0.682cm}{0.097cm}}
\pgfusepath{stroke}
\end{pgfscope}
\end{pgfscope}
\end{pgfscope}
\end{pgfscope}
\end{tikzpicture}}}
   \succ \phi_{M, \varepsilon} ) + 3\lambda b_{M, \varepsilon} \phi_{M,
   \varepsilon} + 3 \llbracket X_{M, \varepsilon}^2 \rrbracket \circ \chi_{M,
   \varepsilon} \]
\[ = - \lambda\tilde{X}_{M, \varepsilon}^{\!\resizebox{!}{.8em}{
\begin{tikzpicture}
\pgfpathmoveto{\pgfqpoint{0cm}{-0.035cm}}
\pgfpathlineto{\pgfqpoint{1.976cm}{-0.035cm}}
\pgfpathlineto{\pgfqpoint{1.976cm}{1.94cm}}
\pgfpathlineto{\pgfqpoint{0cm}{1.94cm}}
\pgfpathclose
\pgfusepath{clip}
\begin{pgfscope}
\begin{pgfscope}
\pgfpathmoveto{\pgfqpoint{0cm}{-0.035cm}}
\pgfpathlineto{\pgfqpoint{1.976cm}{-0.035cm}}
\pgfpathlineto{\pgfqpoint{1.976cm}{1.94cm}}
\pgfpathlineto{\pgfqpoint{0cm}{1.94cm}}
\pgfpathclose
\pgfusepath{clip}
\begin{pgfscope}
\begin{pgfscope}
\pgfsetdash{}{0cm}
\pgfsetlinewidth{0.818mm}
\pgfsetroundcap
\pgfsetroundjoin
\pgfsetmiterlimit{7.0}
\definecolor{eps2pgf_color}{gray}{0}\pgfsetstrokecolor{eps2pgf_color}\pgfsetfillcolor{eps2pgf_color}
\pgfpathmoveto{\pgfqpoint{0.117cm}{1.815cm}}
\pgfpathlineto{\pgfqpoint{0.682cm}{1.065cm}}
\pgfpathlineto{\pgfqpoint{1.246cm}{1.815cm}}
\pgfusepath{stroke}
\end{pgfscope}
\definecolor{eps2pgf_color}{gray}{0}\pgfsetstrokecolor{eps2pgf_color}\pgfsetfillcolor{eps2pgf_color}
\pgfpathmoveto{\pgfqpoint{0.273cm}{1.789cm}}
\pgfpathcurveto{\pgfqpoint{0.273cm}{1.825cm}}{\pgfqpoint{0.259cm}{1.86cm}}{\pgfqpoint{0.233cm}{1.886cm}}
\pgfpathcurveto{\pgfqpoint{0.207cm}{1.912cm}}{\pgfqpoint{0.173cm}{1.926cm}}{\pgfqpoint{0.137cm}{1.926cm}}
\pgfpathcurveto{\pgfqpoint{0.1cm}{1.926cm}}{\pgfqpoint{0.066cm}{1.912cm}}{\pgfqpoint{0.04cm}{1.886cm}}
\pgfpathcurveto{\pgfqpoint{0.014cm}{1.86cm}}{\pgfqpoint{0cm}{1.825cm}}{\pgfqpoint{0cm}{1.789cm}}
\pgfpathcurveto{\pgfqpoint{0cm}{1.753cm}}{\pgfqpoint{0.014cm}{1.718cm}}{\pgfqpoint{0.04cm}{1.692cm}}
\pgfpathcurveto{\pgfqpoint{0.066cm}{1.667cm}}{\pgfqpoint{0.1cm}{1.652cm}}{\pgfqpoint{0.137cm}{1.652cm}}
\pgfpathcurveto{\pgfqpoint{0.173cm}{1.652cm}}{\pgfqpoint{0.207cm}{1.667cm}}{\pgfqpoint{0.233cm}{1.692cm}}
\pgfpathcurveto{\pgfqpoint{0.259cm}{1.718cm}}{\pgfqpoint{0.273cm}{1.753cm}}{\pgfqpoint{0.273cm}{1.789cm}}
\pgfusepath{fill}
\pgfpathmoveto{\pgfqpoint{1.345cm}{1.765cm}}
\pgfpathcurveto{\pgfqpoint{1.345cm}{1.801cm}}{\pgfqpoint{1.331cm}{1.836cm}}{\pgfqpoint{1.305cm}{1.862cm}}
\pgfpathcurveto{\pgfqpoint{1.28cm}{1.887cm}}{\pgfqpoint{1.245cm}{1.902cm}}{\pgfqpoint{1.209cm}{1.902cm}}
\pgfpathcurveto{\pgfqpoint{1.172cm}{1.902cm}}{\pgfqpoint{1.138cm}{1.887cm}}{\pgfqpoint{1.112cm}{1.862cm}}
\pgfpathcurveto{\pgfqpoint{1.087cm}{1.836cm}}{\pgfqpoint{1.072cm}{1.801cm}}{\pgfqpoint{1.072cm}{1.765cm}}
\pgfpathcurveto{\pgfqpoint{1.072cm}{1.728cm}}{\pgfqpoint{1.087cm}{1.694cm}}{\pgfqpoint{1.112cm}{1.668cm}}
\pgfpathcurveto{\pgfqpoint{1.138cm}{1.642cm}}{\pgfqpoint{1.172cm}{1.628cm}}{\pgfqpoint{1.209cm}{1.628cm}}
\pgfpathcurveto{\pgfqpoint{1.245cm}{1.628cm}}{\pgfqpoint{1.28cm}{1.642cm}}{\pgfqpoint{1.305cm}{1.668cm}}
\pgfpathcurveto{\pgfqpoint{1.331cm}{1.694cm}}{\pgfqpoint{1.345cm}{1.728cm}}{\pgfqpoint{1.345cm}{1.765cm}}
\pgfusepath{fill}
\begin{pgfscope}
\pgfsetdash{}{0cm}
\pgfsetlinewidth{0.818mm}
\pgfsetroundcap
\pgfsetroundjoin
\pgfsetmiterlimit{7.0}
\pgfpathmoveto{\pgfqpoint{0.682cm}{1.065cm}}
\pgfpathlineto{\pgfqpoint{1.246cm}{0.315cm}}
\pgfpathlineto{\pgfqpoint{1.811cm}{1.065cm}}
\pgfusepath{stroke}
\end{pgfscope}
\pgfpathmoveto{\pgfqpoint{1.948cm}{1.065cm}}
\pgfpathcurveto{\pgfqpoint{1.948cm}{1.101cm}}{\pgfqpoint{1.933cm}{1.136cm}}{\pgfqpoint{1.907cm}{1.162cm}}
\pgfpathcurveto{\pgfqpoint{1.882cm}{1.187cm}}{\pgfqpoint{1.847cm}{1.202cm}}{\pgfqpoint{1.811cm}{1.202cm}}
\pgfpathcurveto{\pgfqpoint{1.775cm}{1.202cm}}{\pgfqpoint{1.74cm}{1.187cm}}{\pgfqpoint{1.714cm}{1.162cm}}
\pgfpathcurveto{\pgfqpoint{1.689cm}{1.136cm}}{\pgfqpoint{1.674cm}{1.101cm}}{\pgfqpoint{1.674cm}{1.065cm}}
\pgfpathcurveto{\pgfqpoint{1.674cm}{1.029cm}}{\pgfqpoint{1.689cm}{0.994cm}}{\pgfqpoint{1.714cm}{0.968cm}}
\pgfpathcurveto{\pgfqpoint{1.74cm}{0.942cm}}{\pgfqpoint{1.775cm}{0.928cm}}{\pgfqpoint{1.811cm}{0.928cm}}
\pgfpathcurveto{\pgfqpoint{1.847cm}{0.928cm}}{\pgfqpoint{1.882cm}{0.942cm}}{\pgfqpoint{1.907cm}{0.968cm}}
\pgfpathcurveto{\pgfqpoint{1.933cm}{0.994cm}}{\pgfqpoint{1.948cm}{1.029cm}}{\pgfqpoint{1.948cm}{1.065cm}}
\pgfusepath{fill}
\begin{pgfscope}
\pgfsetdash{}{0cm}
\pgfsetlinewidth{0.818mm}
\pgfsetmiterlimit{7.0}
\pgfpathmoveto{\pgfqpoint{1.246cm}{0.315cm}}
\pgfpathlineto{\pgfqpoint{1.244cm}{1.061cm}}
\pgfusepath{stroke}
\end{pgfscope}
\pgfpathmoveto{\pgfqpoint{1.38cm}{1.065cm}}
\pgfpathcurveto{\pgfqpoint{1.38cm}{1.101cm}}{\pgfqpoint{1.366cm}{1.136cm}}{\pgfqpoint{1.34cm}{1.162cm}}
\pgfpathcurveto{\pgfqpoint{1.315cm}{1.187cm}}{\pgfqpoint{1.28cm}{1.202cm}}{\pgfqpoint{1.244cm}{1.202cm}}
\pgfpathcurveto{\pgfqpoint{1.207cm}{1.202cm}}{\pgfqpoint{1.173cm}{1.187cm}}{\pgfqpoint{1.147cm}{1.162cm}}
\pgfpathcurveto{\pgfqpoint{1.121cm}{1.136cm}}{\pgfqpoint{1.107cm}{1.101cm}}{\pgfqpoint{1.107cm}{1.065cm}}
\pgfpathcurveto{\pgfqpoint{1.107cm}{1.029cm}}{\pgfqpoint{1.121cm}{0.994cm}}{\pgfqpoint{1.147cm}{0.968cm}}
\pgfpathcurveto{\pgfqpoint{1.173cm}{0.942cm}}{\pgfqpoint{1.207cm}{0.928cm}}{\pgfqpoint{1.244cm}{0.928cm}}
\pgfpathcurveto{\pgfqpoint{1.28cm}{0.928cm}}{\pgfqpoint{1.315cm}{0.942cm}}{\pgfqpoint{1.34cm}{0.968cm}}
\pgfpathcurveto{\pgfqpoint{1.366cm}{0.994cm}}{\pgfqpoint{1.38cm}{1.029cm}}{\pgfqpoint{1.38cm}{1.065cm}}
\pgfusepath{fill}
\begin{pgfscope}
\pgfsetdash{}{0cm}
\pgfsetlinewidth{0.818mm}
\pgfsetmiterlimit{4.0}
\pgfpathmoveto{\pgfqpoint{1.383cm}{0.178cm}}
\pgfpathcurveto{\pgfqpoint{1.383cm}{0.214cm}}{\pgfqpoint{1.369cm}{0.249cm}}{\pgfqpoint{1.343cm}{0.275cm}}
\pgfpathcurveto{\pgfqpoint{1.317cm}{0.3cm}}{\pgfqpoint{1.283cm}{0.315cm}}{\pgfqpoint{1.246cm}{0.315cm}}
\pgfpathcurveto{\pgfqpoint{1.21cm}{0.315cm}}{\pgfqpoint{1.175cm}{0.3cm}}{\pgfqpoint{1.15cm}{0.275cm}}
\pgfpathcurveto{\pgfqpoint{1.124cm}{0.249cm}}{\pgfqpoint{1.11cm}{0.214cm}}{\pgfqpoint{1.11cm}{0.178cm}}
\pgfpathcurveto{\pgfqpoint{1.11cm}{0.141cm}}{\pgfqpoint{1.124cm}{0.107cm}}{\pgfqpoint{1.15cm}{0.081cm}}
\pgfpathcurveto{\pgfqpoint{1.175cm}{0.055cm}}{\pgfqpoint{1.21cm}{0.041cm}}{\pgfqpoint{1.246cm}{0.041cm}}
\pgfpathcurveto{\pgfqpoint{1.283cm}{0.041cm}}{\pgfqpoint{1.317cm}{0.055cm}}{\pgfqpoint{1.343cm}{0.081cm}}
\pgfpathcurveto{\pgfqpoint{1.369cm}{0.107cm}}{\pgfqpoint{1.383cm}{0.141cm}}{\pgfqpoint{1.383cm}{0.178cm}}
\pgfusepath{stroke}
\end{pgfscope}
\end{pgfscope}
\end{pgfscope}
\end{pgfscope}
\end{tikzpicture}}} \phi_{M, \varepsilon} + 3\lambda
   (b_{M, \varepsilon} - \tilde{b}_{M, \varepsilon} (t)) \phi_{M, \varepsilon}
   +\lambda C_{\varepsilon} (\phi_{M, \varepsilon}, - 3 X_{M, \varepsilon}^{\!\resizebox{0.6em}{!}{
\begin{tikzpicture}
\pgfpathmoveto{\pgfqpoint{0cm}{0cm}}
\pgfpathlineto{\pgfqpoint{1.376cm}{0cm}}
\pgfpathlineto{\pgfqpoint{1.376cm}{1.588cm}}
\pgfpathlineto{\pgfqpoint{0cm}{1.588cm}}
\pgfpathclose
\pgfusepath{clip}
\begin{pgfscope}
\begin{pgfscope}
\pgfpathmoveto{\pgfqpoint{0cm}{0cm}}
\pgfpathlineto{\pgfqpoint{1.376cm}{0cm}}
\pgfpathlineto{\pgfqpoint{1.376cm}{1.588cm}}
\pgfpathlineto{\pgfqpoint{0cm}{1.588cm}}
\pgfpathclose
\pgfusepath{clip}
\begin{pgfscope}
\begin{pgfscope}
\definecolor{eps2pgf_color}{gray}{0.976471}\pgfsetstrokecolor{eps2pgf_color}\pgfsetfillcolor{eps2pgf_color}
\pgfpathmoveto{\pgfqpoint{0cm}{0cm}}
\pgfpathlineto{\pgfqpoint{1.376cm}{0cm}}
\pgfpathlineto{\pgfqpoint{1.376cm}{1.588cm}}
\pgfpathlineto{\pgfqpoint{0cm}{1.588cm}}
\pgfpathclose
\pgfusepath{fill}
\end{pgfscope}
\begin{pgfscope}
\pgfsetdash{}{0cm}
\pgfsetlinewidth{0.818mm}
\pgfsetroundcap
\pgfsetroundjoin
\pgfsetmiterlimit{7.0}
\definecolor{eps2pgf_color}{gray}{0}\pgfsetstrokecolor{eps2pgf_color}\pgfsetfillcolor{eps2pgf_color}
\pgfpathmoveto{\pgfqpoint{0.117cm}{1.476cm}}
\pgfpathlineto{\pgfqpoint{0.682cm}{0.726cm}}
\pgfpathlineto{\pgfqpoint{1.246cm}{1.476cm}}
\pgfusepath{stroke}
\end{pgfscope}
\definecolor{eps2pgf_color}{gray}{0}\pgfsetstrokecolor{eps2pgf_color}\pgfsetfillcolor{eps2pgf_color}
\pgfpathmoveto{\pgfqpoint{0.273cm}{1.451cm}}
\pgfpathcurveto{\pgfqpoint{0.273cm}{1.487cm}}{\pgfqpoint{0.259cm}{1.522cm}}{\pgfqpoint{0.233cm}{1.547cm}}
\pgfpathcurveto{\pgfqpoint{0.207cm}{1.573cm}}{\pgfqpoint{0.173cm}{1.588cm}}{\pgfqpoint{0.137cm}{1.588cm}}
\pgfpathcurveto{\pgfqpoint{0.1cm}{1.588cm}}{\pgfqpoint{0.066cm}{1.573cm}}{\pgfqpoint{0.04cm}{1.547cm}}
\pgfpathcurveto{\pgfqpoint{0.014cm}{1.522cm}}{\pgfqpoint{0cm}{1.487cm}}{\pgfqpoint{0cm}{1.451cm}}
\pgfpathcurveto{\pgfqpoint{0cm}{1.414cm}}{\pgfqpoint{0.014cm}{1.379cm}}{\pgfqpoint{0.04cm}{1.354cm}}
\pgfpathcurveto{\pgfqpoint{0.066cm}{1.328cm}}{\pgfqpoint{0.1cm}{1.314cm}}{\pgfqpoint{0.137cm}{1.314cm}}
\pgfpathcurveto{\pgfqpoint{0.173cm}{1.314cm}}{\pgfqpoint{0.207cm}{1.328cm}}{\pgfqpoint{0.233cm}{1.354cm}}
\pgfpathcurveto{\pgfqpoint{0.259cm}{1.379cm}}{\pgfqpoint{0.273cm}{1.414cm}}{\pgfqpoint{0.273cm}{1.451cm}}
\pgfusepath{fill}
\pgfpathmoveto{\pgfqpoint{1.345cm}{1.426cm}}
\pgfpathcurveto{\pgfqpoint{1.345cm}{1.463cm}}{\pgfqpoint{1.331cm}{1.497cm}}{\pgfqpoint{1.305cm}{1.523cm}}
\pgfpathcurveto{\pgfqpoint{1.28cm}{1.549cm}}{\pgfqpoint{1.245cm}{1.563cm}}{\pgfqpoint{1.209cm}{1.563cm}}
\pgfpathcurveto{\pgfqpoint{1.172cm}{1.563cm}}{\pgfqpoint{1.138cm}{1.549cm}}{\pgfqpoint{1.112cm}{1.523cm}}
\pgfpathcurveto{\pgfqpoint{1.087cm}{1.497cm}}{\pgfqpoint{1.072cm}{1.463cm}}{\pgfqpoint{1.072cm}{1.426cm}}
\pgfpathcurveto{\pgfqpoint{1.072cm}{1.39cm}}{\pgfqpoint{1.087cm}{1.355cm}}{\pgfqpoint{1.112cm}{1.329cm}}
\pgfpathcurveto{\pgfqpoint{1.138cm}{1.304cm}}{\pgfqpoint{1.172cm}{1.289cm}}{\pgfqpoint{1.209cm}{1.289cm}}
\pgfpathcurveto{\pgfqpoint{1.245cm}{1.289cm}}{\pgfqpoint{1.28cm}{1.304cm}}{\pgfqpoint{1.305cm}{1.329cm}}
\pgfpathcurveto{\pgfqpoint{1.331cm}{1.355cm}}{\pgfqpoint{1.345cm}{1.39cm}}{\pgfqpoint{1.345cm}{1.426cm}}
\pgfusepath{fill}
\begin{pgfscope}
\pgfsetdash{}{0cm}
\pgfsetlinewidth{0.818mm}
\pgfsetroundcap
\pgfsetmiterlimit{4.0}
\pgfpathmoveto{\pgfqpoint{0.682cm}{0.726cm}}
\pgfpathlineto{\pgfqpoint{0.682cm}{0.097cm}}
\pgfusepath{stroke}
\end{pgfscope}
\end{pgfscope}
\end{pgfscope}
\end{pgfscope}
\end{tikzpicture}}},
   3 \llbracket X_{M, \varepsilon}^2 \rrbracket) + 3 \llbracket X_{M,
   \varepsilon}^2 \rrbracket \circ \chi_{M, \varepsilon} . \]
The remaining resonant product that requires a decomposition can be treated as
\[ \begin{array}{ccl}
     3 X_{M, \varepsilon} \circ (-\lambda X_{M, \varepsilon}^{\!\resizebox{0.6em}{!}{
\begin{tikzpicture}
\pgfpathmoveto{\pgfqpoint{0cm}{-0.035cm}}
\pgfpathlineto{\pgfqpoint{1.376cm}{-0.035cm}}
\pgfpathlineto{\pgfqpoint{1.376cm}{1.552cm}}
\pgfpathlineto{\pgfqpoint{0cm}{1.552cm}}
\pgfpathclose
\pgfusepath{clip}
\begin{pgfscope}
\begin{pgfscope}
\pgfpathmoveto{\pgfqpoint{0cm}{-0.035cm}}
\pgfpathlineto{\pgfqpoint{1.376cm}{-0.035cm}}
\pgfpathlineto{\pgfqpoint{1.376cm}{1.552cm}}
\pgfpathlineto{\pgfqpoint{0cm}{1.552cm}}
\pgfpathclose
\pgfusepath{clip}
\begin{pgfscope}
\begin{pgfscope}
\pgfsetdash{}{0cm}
\pgfsetlinewidth{0.818mm}
\pgfsetroundcap
\pgfsetroundjoin
\pgfsetmiterlimit{7.0}
\definecolor{eps2pgf_color}{gray}{0}\pgfsetstrokecolor{eps2pgf_color}\pgfsetfillcolor{eps2pgf_color}
\pgfpathmoveto{\pgfqpoint{0.117cm}{1.421cm}}
\pgfpathlineto{\pgfqpoint{0.682cm}{0.671cm}}
\pgfpathlineto{\pgfqpoint{1.246cm}{1.421cm}}
\pgfusepath{stroke}
\end{pgfscope}
\definecolor{eps2pgf_color}{gray}{0}\pgfsetstrokecolor{eps2pgf_color}\pgfsetfillcolor{eps2pgf_color}
\pgfpathmoveto{\pgfqpoint{0.273cm}{1.395cm}}
\pgfpathcurveto{\pgfqpoint{0.273cm}{1.432cm}}{\pgfqpoint{0.259cm}{1.467cm}}{\pgfqpoint{0.233cm}{1.492cm}}
\pgfpathcurveto{\pgfqpoint{0.207cm}{1.518cm}}{\pgfqpoint{0.173cm}{1.532cm}}{\pgfqpoint{0.137cm}{1.532cm}}
\pgfpathcurveto{\pgfqpoint{0.1cm}{1.532cm}}{\pgfqpoint{0.066cm}{1.518cm}}{\pgfqpoint{0.04cm}{1.492cm}}
\pgfpathcurveto{\pgfqpoint{0.014cm}{1.467cm}}{\pgfqpoint{0cm}{1.432cm}}{\pgfqpoint{0cm}{1.395cm}}
\pgfpathcurveto{\pgfqpoint{0cm}{1.359cm}}{\pgfqpoint{0.014cm}{1.324cm}}{\pgfqpoint{0.04cm}{1.299cm}}
\pgfpathcurveto{\pgfqpoint{0.066cm}{1.273cm}}{\pgfqpoint{0.1cm}{1.258cm}}{\pgfqpoint{0.137cm}{1.258cm}}
\pgfpathcurveto{\pgfqpoint{0.173cm}{1.258cm}}{\pgfqpoint{0.207cm}{1.273cm}}{\pgfqpoint{0.233cm}{1.299cm}}
\pgfpathcurveto{\pgfqpoint{0.259cm}{1.324cm}}{\pgfqpoint{0.273cm}{1.359cm}}{\pgfqpoint{0.273cm}{1.395cm}}
\pgfusepath{fill}
\begin{pgfscope}
\pgfsetdash{}{0cm}
\pgfsetlinewidth{0.818mm}
\pgfsetmiterlimit{7.0}
\pgfpathmoveto{\pgfqpoint{0.682cm}{0.671cm}}
\pgfpathlineto{\pgfqpoint{0.679cm}{1.418cm}}
\pgfusepath{stroke}
\end{pgfscope}
\pgfpathmoveto{\pgfqpoint{0.815cm}{1.399cm}}
\pgfpathcurveto{\pgfqpoint{0.815cm}{1.435cm}}{\pgfqpoint{0.801cm}{1.47cm}}{\pgfqpoint{0.775cm}{1.496cm}}
\pgfpathcurveto{\pgfqpoint{0.75cm}{1.521cm}}{\pgfqpoint{0.715cm}{1.536cm}}{\pgfqpoint{0.679cm}{1.536cm}}
\pgfpathcurveto{\pgfqpoint{0.643cm}{1.536cm}}{\pgfqpoint{0.608cm}{1.521cm}}{\pgfqpoint{0.582cm}{1.496cm}}
\pgfpathcurveto{\pgfqpoint{0.557cm}{1.47cm}}{\pgfqpoint{0.542cm}{1.435cm}}{\pgfqpoint{0.542cm}{1.399cm}}
\pgfpathcurveto{\pgfqpoint{0.542cm}{1.363cm}}{\pgfqpoint{0.557cm}{1.328cm}}{\pgfqpoint{0.582cm}{1.302cm}}
\pgfpathcurveto{\pgfqpoint{0.608cm}{1.276cm}}{\pgfqpoint{0.643cm}{1.262cm}}{\pgfqpoint{0.679cm}{1.262cm}}
\pgfpathcurveto{\pgfqpoint{0.715cm}{1.262cm}}{\pgfqpoint{0.75cm}{1.276cm}}{\pgfqpoint{0.775cm}{1.302cm}}
\pgfpathcurveto{\pgfqpoint{0.801cm}{1.328cm}}{\pgfqpoint{0.815cm}{1.363cm}}{\pgfqpoint{0.815cm}{1.399cm}}
\pgfusepath{fill}
\pgfpathmoveto{\pgfqpoint{1.345cm}{1.371cm}}
\pgfpathcurveto{\pgfqpoint{1.345cm}{1.408cm}}{\pgfqpoint{1.331cm}{1.442cm}}{\pgfqpoint{1.305cm}{1.468cm}}
\pgfpathcurveto{\pgfqpoint{1.28cm}{1.494cm}}{\pgfqpoint{1.245cm}{1.508cm}}{\pgfqpoint{1.209cm}{1.508cm}}
\pgfpathcurveto{\pgfqpoint{1.172cm}{1.508cm}}{\pgfqpoint{1.138cm}{1.494cm}}{\pgfqpoint{1.112cm}{1.468cm}}
\pgfpathcurveto{\pgfqpoint{1.087cm}{1.442cm}}{\pgfqpoint{1.072cm}{1.408cm}}{\pgfqpoint{1.072cm}{1.371cm}}
\pgfpathcurveto{\pgfqpoint{1.072cm}{1.335cm}}{\pgfqpoint{1.087cm}{1.3cm}}{\pgfqpoint{1.112cm}{1.274cm}}
\pgfpathcurveto{\pgfqpoint{1.138cm}{1.249cm}}{\pgfqpoint{1.172cm}{1.234cm}}{\pgfqpoint{1.209cm}{1.234cm}}
\pgfpathcurveto{\pgfqpoint{1.245cm}{1.234cm}}{\pgfqpoint{1.28cm}{1.249cm}}{\pgfqpoint{1.305cm}{1.274cm}}
\pgfpathcurveto{\pgfqpoint{1.331cm}{1.3cm}}{\pgfqpoint{1.345cm}{1.335cm}}{\pgfqpoint{1.345cm}{1.371cm}}
\pgfusepath{fill}
\begin{pgfscope}
\pgfsetdash{}{0cm}
\pgfsetlinewidth{0.818mm}
\pgfsetroundcap
\pgfsetmiterlimit{4.0}
\pgfpathmoveto{\pgfqpoint{0.682cm}{0.671cm}}
\pgfpathlineto{\pgfqpoint{0.682cm}{0.042cm}}
\pgfusepath{stroke}
\end{pgfscope}
\end{pgfscope}
\end{pgfscope}
\end{pgfscope}
\end{tikzpicture}}} + \zeta_{M,
     \varepsilon})^2 & = & 3 \lambda^2 X_{M, \varepsilon} \circ (X_{M,
     \varepsilon}^{\!\resizebox{0.6em}{!}{
\begin{tikzpicture}
\pgfpathmoveto{\pgfqpoint{0cm}{-0.035cm}}
\pgfpathlineto{\pgfqpoint{1.376cm}{-0.035cm}}
\pgfpathlineto{\pgfqpoint{1.376cm}{1.552cm}}
\pgfpathlineto{\pgfqpoint{0cm}{1.552cm}}
\pgfpathclose
\pgfusepath{clip}
\begin{pgfscope}
\begin{pgfscope}
\pgfpathmoveto{\pgfqpoint{0cm}{-0.035cm}}
\pgfpathlineto{\pgfqpoint{1.376cm}{-0.035cm}}
\pgfpathlineto{\pgfqpoint{1.376cm}{1.552cm}}
\pgfpathlineto{\pgfqpoint{0cm}{1.552cm}}
\pgfpathclose
\pgfusepath{clip}
\begin{pgfscope}
\begin{pgfscope}
\pgfsetdash{}{0cm}
\pgfsetlinewidth{0.818mm}
\pgfsetroundcap
\pgfsetroundjoin
\pgfsetmiterlimit{7.0}
\definecolor{eps2pgf_color}{gray}{0}\pgfsetstrokecolor{eps2pgf_color}\pgfsetfillcolor{eps2pgf_color}
\pgfpathmoveto{\pgfqpoint{0.117cm}{1.421cm}}
\pgfpathlineto{\pgfqpoint{0.682cm}{0.671cm}}
\pgfpathlineto{\pgfqpoint{1.246cm}{1.421cm}}
\pgfusepath{stroke}
\end{pgfscope}
\definecolor{eps2pgf_color}{gray}{0}\pgfsetstrokecolor{eps2pgf_color}\pgfsetfillcolor{eps2pgf_color}
\pgfpathmoveto{\pgfqpoint{0.273cm}{1.395cm}}
\pgfpathcurveto{\pgfqpoint{0.273cm}{1.432cm}}{\pgfqpoint{0.259cm}{1.467cm}}{\pgfqpoint{0.233cm}{1.492cm}}
\pgfpathcurveto{\pgfqpoint{0.207cm}{1.518cm}}{\pgfqpoint{0.173cm}{1.532cm}}{\pgfqpoint{0.137cm}{1.532cm}}
\pgfpathcurveto{\pgfqpoint{0.1cm}{1.532cm}}{\pgfqpoint{0.066cm}{1.518cm}}{\pgfqpoint{0.04cm}{1.492cm}}
\pgfpathcurveto{\pgfqpoint{0.014cm}{1.467cm}}{\pgfqpoint{0cm}{1.432cm}}{\pgfqpoint{0cm}{1.395cm}}
\pgfpathcurveto{\pgfqpoint{0cm}{1.359cm}}{\pgfqpoint{0.014cm}{1.324cm}}{\pgfqpoint{0.04cm}{1.299cm}}
\pgfpathcurveto{\pgfqpoint{0.066cm}{1.273cm}}{\pgfqpoint{0.1cm}{1.258cm}}{\pgfqpoint{0.137cm}{1.258cm}}
\pgfpathcurveto{\pgfqpoint{0.173cm}{1.258cm}}{\pgfqpoint{0.207cm}{1.273cm}}{\pgfqpoint{0.233cm}{1.299cm}}
\pgfpathcurveto{\pgfqpoint{0.259cm}{1.324cm}}{\pgfqpoint{0.273cm}{1.359cm}}{\pgfqpoint{0.273cm}{1.395cm}}
\pgfusepath{fill}
\begin{pgfscope}
\pgfsetdash{}{0cm}
\pgfsetlinewidth{0.818mm}
\pgfsetmiterlimit{7.0}
\pgfpathmoveto{\pgfqpoint{0.682cm}{0.671cm}}
\pgfpathlineto{\pgfqpoint{0.679cm}{1.418cm}}
\pgfusepath{stroke}
\end{pgfscope}
\pgfpathmoveto{\pgfqpoint{0.815cm}{1.399cm}}
\pgfpathcurveto{\pgfqpoint{0.815cm}{1.435cm}}{\pgfqpoint{0.801cm}{1.47cm}}{\pgfqpoint{0.775cm}{1.496cm}}
\pgfpathcurveto{\pgfqpoint{0.75cm}{1.521cm}}{\pgfqpoint{0.715cm}{1.536cm}}{\pgfqpoint{0.679cm}{1.536cm}}
\pgfpathcurveto{\pgfqpoint{0.643cm}{1.536cm}}{\pgfqpoint{0.608cm}{1.521cm}}{\pgfqpoint{0.582cm}{1.496cm}}
\pgfpathcurveto{\pgfqpoint{0.557cm}{1.47cm}}{\pgfqpoint{0.542cm}{1.435cm}}{\pgfqpoint{0.542cm}{1.399cm}}
\pgfpathcurveto{\pgfqpoint{0.542cm}{1.363cm}}{\pgfqpoint{0.557cm}{1.328cm}}{\pgfqpoint{0.582cm}{1.302cm}}
\pgfpathcurveto{\pgfqpoint{0.608cm}{1.276cm}}{\pgfqpoint{0.643cm}{1.262cm}}{\pgfqpoint{0.679cm}{1.262cm}}
\pgfpathcurveto{\pgfqpoint{0.715cm}{1.262cm}}{\pgfqpoint{0.75cm}{1.276cm}}{\pgfqpoint{0.775cm}{1.302cm}}
\pgfpathcurveto{\pgfqpoint{0.801cm}{1.328cm}}{\pgfqpoint{0.815cm}{1.363cm}}{\pgfqpoint{0.815cm}{1.399cm}}
\pgfusepath{fill}
\pgfpathmoveto{\pgfqpoint{1.345cm}{1.371cm}}
\pgfpathcurveto{\pgfqpoint{1.345cm}{1.408cm}}{\pgfqpoint{1.331cm}{1.442cm}}{\pgfqpoint{1.305cm}{1.468cm}}
\pgfpathcurveto{\pgfqpoint{1.28cm}{1.494cm}}{\pgfqpoint{1.245cm}{1.508cm}}{\pgfqpoint{1.209cm}{1.508cm}}
\pgfpathcurveto{\pgfqpoint{1.172cm}{1.508cm}}{\pgfqpoint{1.138cm}{1.494cm}}{\pgfqpoint{1.112cm}{1.468cm}}
\pgfpathcurveto{\pgfqpoint{1.087cm}{1.442cm}}{\pgfqpoint{1.072cm}{1.408cm}}{\pgfqpoint{1.072cm}{1.371cm}}
\pgfpathcurveto{\pgfqpoint{1.072cm}{1.335cm}}{\pgfqpoint{1.087cm}{1.3cm}}{\pgfqpoint{1.112cm}{1.274cm}}
\pgfpathcurveto{\pgfqpoint{1.138cm}{1.249cm}}{\pgfqpoint{1.172cm}{1.234cm}}{\pgfqpoint{1.209cm}{1.234cm}}
\pgfpathcurveto{\pgfqpoint{1.245cm}{1.234cm}}{\pgfqpoint{1.28cm}{1.249cm}}{\pgfqpoint{1.305cm}{1.274cm}}
\pgfpathcurveto{\pgfqpoint{1.331cm}{1.3cm}}{\pgfqpoint{1.345cm}{1.335cm}}{\pgfqpoint{1.345cm}{1.371cm}}
\pgfusepath{fill}
\begin{pgfscope}
\pgfsetdash{}{0cm}
\pgfsetlinewidth{0.818mm}
\pgfsetroundcap
\pgfsetmiterlimit{4.0}
\pgfpathmoveto{\pgfqpoint{0.682cm}{0.671cm}}
\pgfpathlineto{\pgfqpoint{0.682cm}{0.042cm}}
\pgfusepath{stroke}
\end{pgfscope}
\end{pgfscope}
\end{pgfscope}
\end{pgfscope}
\end{tikzpicture}}})^2 - 6\lambda X_{M, \varepsilon} \circ ( X^{\!\resizebox{0.6em}{!}{
\begin{tikzpicture}
\pgfpathmoveto{\pgfqpoint{0cm}{-0.035cm}}
\pgfpathlineto{\pgfqpoint{1.376cm}{-0.035cm}}
\pgfpathlineto{\pgfqpoint{1.376cm}{1.552cm}}
\pgfpathlineto{\pgfqpoint{0cm}{1.552cm}}
\pgfpathclose
\pgfusepath{clip}
\begin{pgfscope}
\begin{pgfscope}
\pgfpathmoveto{\pgfqpoint{0cm}{-0.035cm}}
\pgfpathlineto{\pgfqpoint{1.376cm}{-0.035cm}}
\pgfpathlineto{\pgfqpoint{1.376cm}{1.552cm}}
\pgfpathlineto{\pgfqpoint{0cm}{1.552cm}}
\pgfpathclose
\pgfusepath{clip}
\begin{pgfscope}
\begin{pgfscope}
\pgfsetdash{}{0cm}
\pgfsetlinewidth{0.818mm}
\pgfsetroundcap
\pgfsetroundjoin
\pgfsetmiterlimit{7.0}
\definecolor{eps2pgf_color}{gray}{0}\pgfsetstrokecolor{eps2pgf_color}\pgfsetfillcolor{eps2pgf_color}
\pgfpathmoveto{\pgfqpoint{0.117cm}{1.421cm}}
\pgfpathlineto{\pgfqpoint{0.682cm}{0.671cm}}
\pgfpathlineto{\pgfqpoint{1.246cm}{1.421cm}}
\pgfusepath{stroke}
\end{pgfscope}
\definecolor{eps2pgf_color}{gray}{0}\pgfsetstrokecolor{eps2pgf_color}\pgfsetfillcolor{eps2pgf_color}
\pgfpathmoveto{\pgfqpoint{0.273cm}{1.395cm}}
\pgfpathcurveto{\pgfqpoint{0.273cm}{1.432cm}}{\pgfqpoint{0.259cm}{1.467cm}}{\pgfqpoint{0.233cm}{1.492cm}}
\pgfpathcurveto{\pgfqpoint{0.207cm}{1.518cm}}{\pgfqpoint{0.173cm}{1.532cm}}{\pgfqpoint{0.137cm}{1.532cm}}
\pgfpathcurveto{\pgfqpoint{0.1cm}{1.532cm}}{\pgfqpoint{0.066cm}{1.518cm}}{\pgfqpoint{0.04cm}{1.492cm}}
\pgfpathcurveto{\pgfqpoint{0.014cm}{1.467cm}}{\pgfqpoint{0cm}{1.432cm}}{\pgfqpoint{0cm}{1.395cm}}
\pgfpathcurveto{\pgfqpoint{0cm}{1.359cm}}{\pgfqpoint{0.014cm}{1.324cm}}{\pgfqpoint{0.04cm}{1.299cm}}
\pgfpathcurveto{\pgfqpoint{0.066cm}{1.273cm}}{\pgfqpoint{0.1cm}{1.258cm}}{\pgfqpoint{0.137cm}{1.258cm}}
\pgfpathcurveto{\pgfqpoint{0.173cm}{1.258cm}}{\pgfqpoint{0.207cm}{1.273cm}}{\pgfqpoint{0.233cm}{1.299cm}}
\pgfpathcurveto{\pgfqpoint{0.259cm}{1.324cm}}{\pgfqpoint{0.273cm}{1.359cm}}{\pgfqpoint{0.273cm}{1.395cm}}
\pgfusepath{fill}
\begin{pgfscope}
\pgfsetdash{}{0cm}
\pgfsetlinewidth{0.818mm}
\pgfsetmiterlimit{7.0}
\pgfpathmoveto{\pgfqpoint{0.682cm}{0.671cm}}
\pgfpathlineto{\pgfqpoint{0.679cm}{1.418cm}}
\pgfusepath{stroke}
\end{pgfscope}
\pgfpathmoveto{\pgfqpoint{0.815cm}{1.399cm}}
\pgfpathcurveto{\pgfqpoint{0.815cm}{1.435cm}}{\pgfqpoint{0.801cm}{1.47cm}}{\pgfqpoint{0.775cm}{1.496cm}}
\pgfpathcurveto{\pgfqpoint{0.75cm}{1.521cm}}{\pgfqpoint{0.715cm}{1.536cm}}{\pgfqpoint{0.679cm}{1.536cm}}
\pgfpathcurveto{\pgfqpoint{0.643cm}{1.536cm}}{\pgfqpoint{0.608cm}{1.521cm}}{\pgfqpoint{0.582cm}{1.496cm}}
\pgfpathcurveto{\pgfqpoint{0.557cm}{1.47cm}}{\pgfqpoint{0.542cm}{1.435cm}}{\pgfqpoint{0.542cm}{1.399cm}}
\pgfpathcurveto{\pgfqpoint{0.542cm}{1.363cm}}{\pgfqpoint{0.557cm}{1.328cm}}{\pgfqpoint{0.582cm}{1.302cm}}
\pgfpathcurveto{\pgfqpoint{0.608cm}{1.276cm}}{\pgfqpoint{0.643cm}{1.262cm}}{\pgfqpoint{0.679cm}{1.262cm}}
\pgfpathcurveto{\pgfqpoint{0.715cm}{1.262cm}}{\pgfqpoint{0.75cm}{1.276cm}}{\pgfqpoint{0.775cm}{1.302cm}}
\pgfpathcurveto{\pgfqpoint{0.801cm}{1.328cm}}{\pgfqpoint{0.815cm}{1.363cm}}{\pgfqpoint{0.815cm}{1.399cm}}
\pgfusepath{fill}
\pgfpathmoveto{\pgfqpoint{1.345cm}{1.371cm}}
\pgfpathcurveto{\pgfqpoint{1.345cm}{1.408cm}}{\pgfqpoint{1.331cm}{1.442cm}}{\pgfqpoint{1.305cm}{1.468cm}}
\pgfpathcurveto{\pgfqpoint{1.28cm}{1.494cm}}{\pgfqpoint{1.245cm}{1.508cm}}{\pgfqpoint{1.209cm}{1.508cm}}
\pgfpathcurveto{\pgfqpoint{1.172cm}{1.508cm}}{\pgfqpoint{1.138cm}{1.494cm}}{\pgfqpoint{1.112cm}{1.468cm}}
\pgfpathcurveto{\pgfqpoint{1.087cm}{1.442cm}}{\pgfqpoint{1.072cm}{1.408cm}}{\pgfqpoint{1.072cm}{1.371cm}}
\pgfpathcurveto{\pgfqpoint{1.072cm}{1.335cm}}{\pgfqpoint{1.087cm}{1.3cm}}{\pgfqpoint{1.112cm}{1.274cm}}
\pgfpathcurveto{\pgfqpoint{1.138cm}{1.249cm}}{\pgfqpoint{1.172cm}{1.234cm}}{\pgfqpoint{1.209cm}{1.234cm}}
\pgfpathcurveto{\pgfqpoint{1.245cm}{1.234cm}}{\pgfqpoint{1.28cm}{1.249cm}}{\pgfqpoint{1.305cm}{1.274cm}}
\pgfpathcurveto{\pgfqpoint{1.331cm}{1.3cm}}{\pgfqpoint{1.345cm}{1.335cm}}{\pgfqpoint{1.345cm}{1.371cm}}
\pgfusepath{fill}
\begin{pgfscope}
\pgfsetdash{}{0cm}
\pgfsetlinewidth{0.818mm}
\pgfsetroundcap
\pgfsetmiterlimit{4.0}
\pgfpathmoveto{\pgfqpoint{0.682cm}{0.671cm}}
\pgfpathlineto{\pgfqpoint{0.682cm}{0.042cm}}
\pgfusepath{stroke}
\end{pgfscope}
\end{pgfscope}
\end{pgfscope}
\end{pgfscope}
\end{tikzpicture}}}_{M,
     \varepsilon} \zeta_{M, \varepsilon} ) + 3 X_{M, \varepsilon} \circ
     \zeta_{M, \varepsilon}^2\\
     & = & 6\lambda^2 X_{M, \varepsilon} \circ (X_{M, \varepsilon}^{\!\resizebox{0.6em}{!}{
\begin{tikzpicture}
\pgfpathmoveto{\pgfqpoint{0cm}{-0.035cm}}
\pgfpathlineto{\pgfqpoint{1.376cm}{-0.035cm}}
\pgfpathlineto{\pgfqpoint{1.376cm}{1.552cm}}
\pgfpathlineto{\pgfqpoint{0cm}{1.552cm}}
\pgfpathclose
\pgfusepath{clip}
\begin{pgfscope}
\begin{pgfscope}
\pgfpathmoveto{\pgfqpoint{0cm}{-0.035cm}}
\pgfpathlineto{\pgfqpoint{1.376cm}{-0.035cm}}
\pgfpathlineto{\pgfqpoint{1.376cm}{1.552cm}}
\pgfpathlineto{\pgfqpoint{0cm}{1.552cm}}
\pgfpathclose
\pgfusepath{clip}
\begin{pgfscope}
\begin{pgfscope}
\pgfsetdash{}{0cm}
\pgfsetlinewidth{0.818mm}
\pgfsetroundcap
\pgfsetroundjoin
\pgfsetmiterlimit{7.0}
\definecolor{eps2pgf_color}{gray}{0}\pgfsetstrokecolor{eps2pgf_color}\pgfsetfillcolor{eps2pgf_color}
\pgfpathmoveto{\pgfqpoint{0.117cm}{1.421cm}}
\pgfpathlineto{\pgfqpoint{0.682cm}{0.671cm}}
\pgfpathlineto{\pgfqpoint{1.246cm}{1.421cm}}
\pgfusepath{stroke}
\end{pgfscope}
\definecolor{eps2pgf_color}{gray}{0}\pgfsetstrokecolor{eps2pgf_color}\pgfsetfillcolor{eps2pgf_color}
\pgfpathmoveto{\pgfqpoint{0.273cm}{1.395cm}}
\pgfpathcurveto{\pgfqpoint{0.273cm}{1.432cm}}{\pgfqpoint{0.259cm}{1.467cm}}{\pgfqpoint{0.233cm}{1.492cm}}
\pgfpathcurveto{\pgfqpoint{0.207cm}{1.518cm}}{\pgfqpoint{0.173cm}{1.532cm}}{\pgfqpoint{0.137cm}{1.532cm}}
\pgfpathcurveto{\pgfqpoint{0.1cm}{1.532cm}}{\pgfqpoint{0.066cm}{1.518cm}}{\pgfqpoint{0.04cm}{1.492cm}}
\pgfpathcurveto{\pgfqpoint{0.014cm}{1.467cm}}{\pgfqpoint{0cm}{1.432cm}}{\pgfqpoint{0cm}{1.395cm}}
\pgfpathcurveto{\pgfqpoint{0cm}{1.359cm}}{\pgfqpoint{0.014cm}{1.324cm}}{\pgfqpoint{0.04cm}{1.299cm}}
\pgfpathcurveto{\pgfqpoint{0.066cm}{1.273cm}}{\pgfqpoint{0.1cm}{1.258cm}}{\pgfqpoint{0.137cm}{1.258cm}}
\pgfpathcurveto{\pgfqpoint{0.173cm}{1.258cm}}{\pgfqpoint{0.207cm}{1.273cm}}{\pgfqpoint{0.233cm}{1.299cm}}
\pgfpathcurveto{\pgfqpoint{0.259cm}{1.324cm}}{\pgfqpoint{0.273cm}{1.359cm}}{\pgfqpoint{0.273cm}{1.395cm}}
\pgfusepath{fill}
\begin{pgfscope}
\pgfsetdash{}{0cm}
\pgfsetlinewidth{0.818mm}
\pgfsetmiterlimit{7.0}
\pgfpathmoveto{\pgfqpoint{0.682cm}{0.671cm}}
\pgfpathlineto{\pgfqpoint{0.679cm}{1.418cm}}
\pgfusepath{stroke}
\end{pgfscope}
\pgfpathmoveto{\pgfqpoint{0.815cm}{1.399cm}}
\pgfpathcurveto{\pgfqpoint{0.815cm}{1.435cm}}{\pgfqpoint{0.801cm}{1.47cm}}{\pgfqpoint{0.775cm}{1.496cm}}
\pgfpathcurveto{\pgfqpoint{0.75cm}{1.521cm}}{\pgfqpoint{0.715cm}{1.536cm}}{\pgfqpoint{0.679cm}{1.536cm}}
\pgfpathcurveto{\pgfqpoint{0.643cm}{1.536cm}}{\pgfqpoint{0.608cm}{1.521cm}}{\pgfqpoint{0.582cm}{1.496cm}}
\pgfpathcurveto{\pgfqpoint{0.557cm}{1.47cm}}{\pgfqpoint{0.542cm}{1.435cm}}{\pgfqpoint{0.542cm}{1.399cm}}
\pgfpathcurveto{\pgfqpoint{0.542cm}{1.363cm}}{\pgfqpoint{0.557cm}{1.328cm}}{\pgfqpoint{0.582cm}{1.302cm}}
\pgfpathcurveto{\pgfqpoint{0.608cm}{1.276cm}}{\pgfqpoint{0.643cm}{1.262cm}}{\pgfqpoint{0.679cm}{1.262cm}}
\pgfpathcurveto{\pgfqpoint{0.715cm}{1.262cm}}{\pgfqpoint{0.75cm}{1.276cm}}{\pgfqpoint{0.775cm}{1.302cm}}
\pgfpathcurveto{\pgfqpoint{0.801cm}{1.328cm}}{\pgfqpoint{0.815cm}{1.363cm}}{\pgfqpoint{0.815cm}{1.399cm}}
\pgfusepath{fill}
\pgfpathmoveto{\pgfqpoint{1.345cm}{1.371cm}}
\pgfpathcurveto{\pgfqpoint{1.345cm}{1.408cm}}{\pgfqpoint{1.331cm}{1.442cm}}{\pgfqpoint{1.305cm}{1.468cm}}
\pgfpathcurveto{\pgfqpoint{1.28cm}{1.494cm}}{\pgfqpoint{1.245cm}{1.508cm}}{\pgfqpoint{1.209cm}{1.508cm}}
\pgfpathcurveto{\pgfqpoint{1.172cm}{1.508cm}}{\pgfqpoint{1.138cm}{1.494cm}}{\pgfqpoint{1.112cm}{1.468cm}}
\pgfpathcurveto{\pgfqpoint{1.087cm}{1.442cm}}{\pgfqpoint{1.072cm}{1.408cm}}{\pgfqpoint{1.072cm}{1.371cm}}
\pgfpathcurveto{\pgfqpoint{1.072cm}{1.335cm}}{\pgfqpoint{1.087cm}{1.3cm}}{\pgfqpoint{1.112cm}{1.274cm}}
\pgfpathcurveto{\pgfqpoint{1.138cm}{1.249cm}}{\pgfqpoint{1.172cm}{1.234cm}}{\pgfqpoint{1.209cm}{1.234cm}}
\pgfpathcurveto{\pgfqpoint{1.245cm}{1.234cm}}{\pgfqpoint{1.28cm}{1.249cm}}{\pgfqpoint{1.305cm}{1.274cm}}
\pgfpathcurveto{\pgfqpoint{1.331cm}{1.3cm}}{\pgfqpoint{1.345cm}{1.335cm}}{\pgfqpoint{1.345cm}{1.371cm}}
\pgfusepath{fill}
\begin{pgfscope}
\pgfsetdash{}{0cm}
\pgfsetlinewidth{0.818mm}
\pgfsetroundcap
\pgfsetmiterlimit{4.0}
\pgfpathmoveto{\pgfqpoint{0.682cm}{0.671cm}}
\pgfpathlineto{\pgfqpoint{0.682cm}{0.042cm}}
\pgfusepath{stroke}
\end{pgfscope}
\end{pgfscope}
\end{pgfscope}
\end{pgfscope}
\end{tikzpicture}}} \succ
     X_{M, \varepsilon}^{\!\resizebox{0.6em}{!}{
\begin{tikzpicture}
\pgfpathmoveto{\pgfqpoint{0cm}{-0.035cm}}
\pgfpathlineto{\pgfqpoint{1.376cm}{-0.035cm}}
\pgfpathlineto{\pgfqpoint{1.376cm}{1.552cm}}
\pgfpathlineto{\pgfqpoint{0cm}{1.552cm}}
\pgfpathclose
\pgfusepath{clip}
\begin{pgfscope}
\begin{pgfscope}
\pgfpathmoveto{\pgfqpoint{0cm}{-0.035cm}}
\pgfpathlineto{\pgfqpoint{1.376cm}{-0.035cm}}
\pgfpathlineto{\pgfqpoint{1.376cm}{1.552cm}}
\pgfpathlineto{\pgfqpoint{0cm}{1.552cm}}
\pgfpathclose
\pgfusepath{clip}
\begin{pgfscope}
\begin{pgfscope}
\pgfsetdash{}{0cm}
\pgfsetlinewidth{0.818mm}
\pgfsetroundcap
\pgfsetroundjoin
\pgfsetmiterlimit{7.0}
\definecolor{eps2pgf_color}{gray}{0}\pgfsetstrokecolor{eps2pgf_color}\pgfsetfillcolor{eps2pgf_color}
\pgfpathmoveto{\pgfqpoint{0.117cm}{1.421cm}}
\pgfpathlineto{\pgfqpoint{0.682cm}{0.671cm}}
\pgfpathlineto{\pgfqpoint{1.246cm}{1.421cm}}
\pgfusepath{stroke}
\end{pgfscope}
\definecolor{eps2pgf_color}{gray}{0}\pgfsetstrokecolor{eps2pgf_color}\pgfsetfillcolor{eps2pgf_color}
\pgfpathmoveto{\pgfqpoint{0.273cm}{1.395cm}}
\pgfpathcurveto{\pgfqpoint{0.273cm}{1.432cm}}{\pgfqpoint{0.259cm}{1.467cm}}{\pgfqpoint{0.233cm}{1.492cm}}
\pgfpathcurveto{\pgfqpoint{0.207cm}{1.518cm}}{\pgfqpoint{0.173cm}{1.532cm}}{\pgfqpoint{0.137cm}{1.532cm}}
\pgfpathcurveto{\pgfqpoint{0.1cm}{1.532cm}}{\pgfqpoint{0.066cm}{1.518cm}}{\pgfqpoint{0.04cm}{1.492cm}}
\pgfpathcurveto{\pgfqpoint{0.014cm}{1.467cm}}{\pgfqpoint{0cm}{1.432cm}}{\pgfqpoint{0cm}{1.395cm}}
\pgfpathcurveto{\pgfqpoint{0cm}{1.359cm}}{\pgfqpoint{0.014cm}{1.324cm}}{\pgfqpoint{0.04cm}{1.299cm}}
\pgfpathcurveto{\pgfqpoint{0.066cm}{1.273cm}}{\pgfqpoint{0.1cm}{1.258cm}}{\pgfqpoint{0.137cm}{1.258cm}}
\pgfpathcurveto{\pgfqpoint{0.173cm}{1.258cm}}{\pgfqpoint{0.207cm}{1.273cm}}{\pgfqpoint{0.233cm}{1.299cm}}
\pgfpathcurveto{\pgfqpoint{0.259cm}{1.324cm}}{\pgfqpoint{0.273cm}{1.359cm}}{\pgfqpoint{0.273cm}{1.395cm}}
\pgfusepath{fill}
\begin{pgfscope}
\pgfsetdash{}{0cm}
\pgfsetlinewidth{0.818mm}
\pgfsetmiterlimit{7.0}
\pgfpathmoveto{\pgfqpoint{0.682cm}{0.671cm}}
\pgfpathlineto{\pgfqpoint{0.679cm}{1.418cm}}
\pgfusepath{stroke}
\end{pgfscope}
\pgfpathmoveto{\pgfqpoint{0.815cm}{1.399cm}}
\pgfpathcurveto{\pgfqpoint{0.815cm}{1.435cm}}{\pgfqpoint{0.801cm}{1.47cm}}{\pgfqpoint{0.775cm}{1.496cm}}
\pgfpathcurveto{\pgfqpoint{0.75cm}{1.521cm}}{\pgfqpoint{0.715cm}{1.536cm}}{\pgfqpoint{0.679cm}{1.536cm}}
\pgfpathcurveto{\pgfqpoint{0.643cm}{1.536cm}}{\pgfqpoint{0.608cm}{1.521cm}}{\pgfqpoint{0.582cm}{1.496cm}}
\pgfpathcurveto{\pgfqpoint{0.557cm}{1.47cm}}{\pgfqpoint{0.542cm}{1.435cm}}{\pgfqpoint{0.542cm}{1.399cm}}
\pgfpathcurveto{\pgfqpoint{0.542cm}{1.363cm}}{\pgfqpoint{0.557cm}{1.328cm}}{\pgfqpoint{0.582cm}{1.302cm}}
\pgfpathcurveto{\pgfqpoint{0.608cm}{1.276cm}}{\pgfqpoint{0.643cm}{1.262cm}}{\pgfqpoint{0.679cm}{1.262cm}}
\pgfpathcurveto{\pgfqpoint{0.715cm}{1.262cm}}{\pgfqpoint{0.75cm}{1.276cm}}{\pgfqpoint{0.775cm}{1.302cm}}
\pgfpathcurveto{\pgfqpoint{0.801cm}{1.328cm}}{\pgfqpoint{0.815cm}{1.363cm}}{\pgfqpoint{0.815cm}{1.399cm}}
\pgfusepath{fill}
\pgfpathmoveto{\pgfqpoint{1.345cm}{1.371cm}}
\pgfpathcurveto{\pgfqpoint{1.345cm}{1.408cm}}{\pgfqpoint{1.331cm}{1.442cm}}{\pgfqpoint{1.305cm}{1.468cm}}
\pgfpathcurveto{\pgfqpoint{1.28cm}{1.494cm}}{\pgfqpoint{1.245cm}{1.508cm}}{\pgfqpoint{1.209cm}{1.508cm}}
\pgfpathcurveto{\pgfqpoint{1.172cm}{1.508cm}}{\pgfqpoint{1.138cm}{1.494cm}}{\pgfqpoint{1.112cm}{1.468cm}}
\pgfpathcurveto{\pgfqpoint{1.087cm}{1.442cm}}{\pgfqpoint{1.072cm}{1.408cm}}{\pgfqpoint{1.072cm}{1.371cm}}
\pgfpathcurveto{\pgfqpoint{1.072cm}{1.335cm}}{\pgfqpoint{1.087cm}{1.3cm}}{\pgfqpoint{1.112cm}{1.274cm}}
\pgfpathcurveto{\pgfqpoint{1.138cm}{1.249cm}}{\pgfqpoint{1.172cm}{1.234cm}}{\pgfqpoint{1.209cm}{1.234cm}}
\pgfpathcurveto{\pgfqpoint{1.245cm}{1.234cm}}{\pgfqpoint{1.28cm}{1.249cm}}{\pgfqpoint{1.305cm}{1.274cm}}
\pgfpathcurveto{\pgfqpoint{1.331cm}{1.3cm}}{\pgfqpoint{1.345cm}{1.335cm}}{\pgfqpoint{1.345cm}{1.371cm}}
\pgfusepath{fill}
\begin{pgfscope}
\pgfsetdash{}{0cm}
\pgfsetlinewidth{0.818mm}
\pgfsetroundcap
\pgfsetmiterlimit{4.0}
\pgfpathmoveto{\pgfqpoint{0.682cm}{0.671cm}}
\pgfpathlineto{\pgfqpoint{0.682cm}{0.042cm}}
\pgfusepath{stroke}
\end{pgfscope}
\end{pgfscope}
\end{pgfscope}
\end{pgfscope}
\end{tikzpicture}}}) + 3\lambda^2 X_{M, \varepsilon} \circ (
     X_{M, \varepsilon}^{\!\resizebox{0.6em}{!}{
\begin{tikzpicture}
\pgfpathmoveto{\pgfqpoint{0cm}{-0.035cm}}
\pgfpathlineto{\pgfqpoint{1.376cm}{-0.035cm}}
\pgfpathlineto{\pgfqpoint{1.376cm}{1.552cm}}
\pgfpathlineto{\pgfqpoint{0cm}{1.552cm}}
\pgfpathclose
\pgfusepath{clip}
\begin{pgfscope}
\begin{pgfscope}
\pgfpathmoveto{\pgfqpoint{0cm}{-0.035cm}}
\pgfpathlineto{\pgfqpoint{1.376cm}{-0.035cm}}
\pgfpathlineto{\pgfqpoint{1.376cm}{1.552cm}}
\pgfpathlineto{\pgfqpoint{0cm}{1.552cm}}
\pgfpathclose
\pgfusepath{clip}
\begin{pgfscope}
\begin{pgfscope}
\pgfsetdash{}{0cm}
\pgfsetlinewidth{0.818mm}
\pgfsetroundcap
\pgfsetroundjoin
\pgfsetmiterlimit{7.0}
\definecolor{eps2pgf_color}{gray}{0}\pgfsetstrokecolor{eps2pgf_color}\pgfsetfillcolor{eps2pgf_color}
\pgfpathmoveto{\pgfqpoint{0.117cm}{1.421cm}}
\pgfpathlineto{\pgfqpoint{0.682cm}{0.671cm}}
\pgfpathlineto{\pgfqpoint{1.246cm}{1.421cm}}
\pgfusepath{stroke}
\end{pgfscope}
\definecolor{eps2pgf_color}{gray}{0}\pgfsetstrokecolor{eps2pgf_color}\pgfsetfillcolor{eps2pgf_color}
\pgfpathmoveto{\pgfqpoint{0.273cm}{1.395cm}}
\pgfpathcurveto{\pgfqpoint{0.273cm}{1.432cm}}{\pgfqpoint{0.259cm}{1.467cm}}{\pgfqpoint{0.233cm}{1.492cm}}
\pgfpathcurveto{\pgfqpoint{0.207cm}{1.518cm}}{\pgfqpoint{0.173cm}{1.532cm}}{\pgfqpoint{0.137cm}{1.532cm}}
\pgfpathcurveto{\pgfqpoint{0.1cm}{1.532cm}}{\pgfqpoint{0.066cm}{1.518cm}}{\pgfqpoint{0.04cm}{1.492cm}}
\pgfpathcurveto{\pgfqpoint{0.014cm}{1.467cm}}{\pgfqpoint{0cm}{1.432cm}}{\pgfqpoint{0cm}{1.395cm}}
\pgfpathcurveto{\pgfqpoint{0cm}{1.359cm}}{\pgfqpoint{0.014cm}{1.324cm}}{\pgfqpoint{0.04cm}{1.299cm}}
\pgfpathcurveto{\pgfqpoint{0.066cm}{1.273cm}}{\pgfqpoint{0.1cm}{1.258cm}}{\pgfqpoint{0.137cm}{1.258cm}}
\pgfpathcurveto{\pgfqpoint{0.173cm}{1.258cm}}{\pgfqpoint{0.207cm}{1.273cm}}{\pgfqpoint{0.233cm}{1.299cm}}
\pgfpathcurveto{\pgfqpoint{0.259cm}{1.324cm}}{\pgfqpoint{0.273cm}{1.359cm}}{\pgfqpoint{0.273cm}{1.395cm}}
\pgfusepath{fill}
\begin{pgfscope}
\pgfsetdash{}{0cm}
\pgfsetlinewidth{0.818mm}
\pgfsetmiterlimit{7.0}
\pgfpathmoveto{\pgfqpoint{0.682cm}{0.671cm}}
\pgfpathlineto{\pgfqpoint{0.679cm}{1.418cm}}
\pgfusepath{stroke}
\end{pgfscope}
\pgfpathmoveto{\pgfqpoint{0.815cm}{1.399cm}}
\pgfpathcurveto{\pgfqpoint{0.815cm}{1.435cm}}{\pgfqpoint{0.801cm}{1.47cm}}{\pgfqpoint{0.775cm}{1.496cm}}
\pgfpathcurveto{\pgfqpoint{0.75cm}{1.521cm}}{\pgfqpoint{0.715cm}{1.536cm}}{\pgfqpoint{0.679cm}{1.536cm}}
\pgfpathcurveto{\pgfqpoint{0.643cm}{1.536cm}}{\pgfqpoint{0.608cm}{1.521cm}}{\pgfqpoint{0.582cm}{1.496cm}}
\pgfpathcurveto{\pgfqpoint{0.557cm}{1.47cm}}{\pgfqpoint{0.542cm}{1.435cm}}{\pgfqpoint{0.542cm}{1.399cm}}
\pgfpathcurveto{\pgfqpoint{0.542cm}{1.363cm}}{\pgfqpoint{0.557cm}{1.328cm}}{\pgfqpoint{0.582cm}{1.302cm}}
\pgfpathcurveto{\pgfqpoint{0.608cm}{1.276cm}}{\pgfqpoint{0.643cm}{1.262cm}}{\pgfqpoint{0.679cm}{1.262cm}}
\pgfpathcurveto{\pgfqpoint{0.715cm}{1.262cm}}{\pgfqpoint{0.75cm}{1.276cm}}{\pgfqpoint{0.775cm}{1.302cm}}
\pgfpathcurveto{\pgfqpoint{0.801cm}{1.328cm}}{\pgfqpoint{0.815cm}{1.363cm}}{\pgfqpoint{0.815cm}{1.399cm}}
\pgfusepath{fill}
\pgfpathmoveto{\pgfqpoint{1.345cm}{1.371cm}}
\pgfpathcurveto{\pgfqpoint{1.345cm}{1.408cm}}{\pgfqpoint{1.331cm}{1.442cm}}{\pgfqpoint{1.305cm}{1.468cm}}
\pgfpathcurveto{\pgfqpoint{1.28cm}{1.494cm}}{\pgfqpoint{1.245cm}{1.508cm}}{\pgfqpoint{1.209cm}{1.508cm}}
\pgfpathcurveto{\pgfqpoint{1.172cm}{1.508cm}}{\pgfqpoint{1.138cm}{1.494cm}}{\pgfqpoint{1.112cm}{1.468cm}}
\pgfpathcurveto{\pgfqpoint{1.087cm}{1.442cm}}{\pgfqpoint{1.072cm}{1.408cm}}{\pgfqpoint{1.072cm}{1.371cm}}
\pgfpathcurveto{\pgfqpoint{1.072cm}{1.335cm}}{\pgfqpoint{1.087cm}{1.3cm}}{\pgfqpoint{1.112cm}{1.274cm}}
\pgfpathcurveto{\pgfqpoint{1.138cm}{1.249cm}}{\pgfqpoint{1.172cm}{1.234cm}}{\pgfqpoint{1.209cm}{1.234cm}}
\pgfpathcurveto{\pgfqpoint{1.245cm}{1.234cm}}{\pgfqpoint{1.28cm}{1.249cm}}{\pgfqpoint{1.305cm}{1.274cm}}
\pgfpathcurveto{\pgfqpoint{1.331cm}{1.3cm}}{\pgfqpoint{1.345cm}{1.335cm}}{\pgfqpoint{1.345cm}{1.371cm}}
\pgfusepath{fill}
\begin{pgfscope}
\pgfsetdash{}{0cm}
\pgfsetlinewidth{0.818mm}
\pgfsetroundcap
\pgfsetmiterlimit{4.0}
\pgfpathmoveto{\pgfqpoint{0.682cm}{0.671cm}}
\pgfpathlineto{\pgfqpoint{0.682cm}{0.042cm}}
\pgfusepath{stroke}
\end{pgfscope}
\end{pgfscope}
\end{pgfscope}
\end{pgfscope}
\end{tikzpicture}}} \circ X_{M, \varepsilon}^{\!\resizebox{0.6em}{!}{
\begin{tikzpicture}
\pgfpathmoveto{\pgfqpoint{0cm}{-0.035cm}}
\pgfpathlineto{\pgfqpoint{1.376cm}{-0.035cm}}
\pgfpathlineto{\pgfqpoint{1.376cm}{1.552cm}}
\pgfpathlineto{\pgfqpoint{0cm}{1.552cm}}
\pgfpathclose
\pgfusepath{clip}
\begin{pgfscope}
\begin{pgfscope}
\pgfpathmoveto{\pgfqpoint{0cm}{-0.035cm}}
\pgfpathlineto{\pgfqpoint{1.376cm}{-0.035cm}}
\pgfpathlineto{\pgfqpoint{1.376cm}{1.552cm}}
\pgfpathlineto{\pgfqpoint{0cm}{1.552cm}}
\pgfpathclose
\pgfusepath{clip}
\begin{pgfscope}
\begin{pgfscope}
\pgfsetdash{}{0cm}
\pgfsetlinewidth{0.818mm}
\pgfsetroundcap
\pgfsetroundjoin
\pgfsetmiterlimit{7.0}
\definecolor{eps2pgf_color}{gray}{0}\pgfsetstrokecolor{eps2pgf_color}\pgfsetfillcolor{eps2pgf_color}
\pgfpathmoveto{\pgfqpoint{0.117cm}{1.421cm}}
\pgfpathlineto{\pgfqpoint{0.682cm}{0.671cm}}
\pgfpathlineto{\pgfqpoint{1.246cm}{1.421cm}}
\pgfusepath{stroke}
\end{pgfscope}
\definecolor{eps2pgf_color}{gray}{0}\pgfsetstrokecolor{eps2pgf_color}\pgfsetfillcolor{eps2pgf_color}
\pgfpathmoveto{\pgfqpoint{0.273cm}{1.395cm}}
\pgfpathcurveto{\pgfqpoint{0.273cm}{1.432cm}}{\pgfqpoint{0.259cm}{1.467cm}}{\pgfqpoint{0.233cm}{1.492cm}}
\pgfpathcurveto{\pgfqpoint{0.207cm}{1.518cm}}{\pgfqpoint{0.173cm}{1.532cm}}{\pgfqpoint{0.137cm}{1.532cm}}
\pgfpathcurveto{\pgfqpoint{0.1cm}{1.532cm}}{\pgfqpoint{0.066cm}{1.518cm}}{\pgfqpoint{0.04cm}{1.492cm}}
\pgfpathcurveto{\pgfqpoint{0.014cm}{1.467cm}}{\pgfqpoint{0cm}{1.432cm}}{\pgfqpoint{0cm}{1.395cm}}
\pgfpathcurveto{\pgfqpoint{0cm}{1.359cm}}{\pgfqpoint{0.014cm}{1.324cm}}{\pgfqpoint{0.04cm}{1.299cm}}
\pgfpathcurveto{\pgfqpoint{0.066cm}{1.273cm}}{\pgfqpoint{0.1cm}{1.258cm}}{\pgfqpoint{0.137cm}{1.258cm}}
\pgfpathcurveto{\pgfqpoint{0.173cm}{1.258cm}}{\pgfqpoint{0.207cm}{1.273cm}}{\pgfqpoint{0.233cm}{1.299cm}}
\pgfpathcurveto{\pgfqpoint{0.259cm}{1.324cm}}{\pgfqpoint{0.273cm}{1.359cm}}{\pgfqpoint{0.273cm}{1.395cm}}
\pgfusepath{fill}
\begin{pgfscope}
\pgfsetdash{}{0cm}
\pgfsetlinewidth{0.818mm}
\pgfsetmiterlimit{7.0}
\pgfpathmoveto{\pgfqpoint{0.682cm}{0.671cm}}
\pgfpathlineto{\pgfqpoint{0.679cm}{1.418cm}}
\pgfusepath{stroke}
\end{pgfscope}
\pgfpathmoveto{\pgfqpoint{0.815cm}{1.399cm}}
\pgfpathcurveto{\pgfqpoint{0.815cm}{1.435cm}}{\pgfqpoint{0.801cm}{1.47cm}}{\pgfqpoint{0.775cm}{1.496cm}}
\pgfpathcurveto{\pgfqpoint{0.75cm}{1.521cm}}{\pgfqpoint{0.715cm}{1.536cm}}{\pgfqpoint{0.679cm}{1.536cm}}
\pgfpathcurveto{\pgfqpoint{0.643cm}{1.536cm}}{\pgfqpoint{0.608cm}{1.521cm}}{\pgfqpoint{0.582cm}{1.496cm}}
\pgfpathcurveto{\pgfqpoint{0.557cm}{1.47cm}}{\pgfqpoint{0.542cm}{1.435cm}}{\pgfqpoint{0.542cm}{1.399cm}}
\pgfpathcurveto{\pgfqpoint{0.542cm}{1.363cm}}{\pgfqpoint{0.557cm}{1.328cm}}{\pgfqpoint{0.582cm}{1.302cm}}
\pgfpathcurveto{\pgfqpoint{0.608cm}{1.276cm}}{\pgfqpoint{0.643cm}{1.262cm}}{\pgfqpoint{0.679cm}{1.262cm}}
\pgfpathcurveto{\pgfqpoint{0.715cm}{1.262cm}}{\pgfqpoint{0.75cm}{1.276cm}}{\pgfqpoint{0.775cm}{1.302cm}}
\pgfpathcurveto{\pgfqpoint{0.801cm}{1.328cm}}{\pgfqpoint{0.815cm}{1.363cm}}{\pgfqpoint{0.815cm}{1.399cm}}
\pgfusepath{fill}
\pgfpathmoveto{\pgfqpoint{1.345cm}{1.371cm}}
\pgfpathcurveto{\pgfqpoint{1.345cm}{1.408cm}}{\pgfqpoint{1.331cm}{1.442cm}}{\pgfqpoint{1.305cm}{1.468cm}}
\pgfpathcurveto{\pgfqpoint{1.28cm}{1.494cm}}{\pgfqpoint{1.245cm}{1.508cm}}{\pgfqpoint{1.209cm}{1.508cm}}
\pgfpathcurveto{\pgfqpoint{1.172cm}{1.508cm}}{\pgfqpoint{1.138cm}{1.494cm}}{\pgfqpoint{1.112cm}{1.468cm}}
\pgfpathcurveto{\pgfqpoint{1.087cm}{1.442cm}}{\pgfqpoint{1.072cm}{1.408cm}}{\pgfqpoint{1.072cm}{1.371cm}}
\pgfpathcurveto{\pgfqpoint{1.072cm}{1.335cm}}{\pgfqpoint{1.087cm}{1.3cm}}{\pgfqpoint{1.112cm}{1.274cm}}
\pgfpathcurveto{\pgfqpoint{1.138cm}{1.249cm}}{\pgfqpoint{1.172cm}{1.234cm}}{\pgfqpoint{1.209cm}{1.234cm}}
\pgfpathcurveto{\pgfqpoint{1.245cm}{1.234cm}}{\pgfqpoint{1.28cm}{1.249cm}}{\pgfqpoint{1.305cm}{1.274cm}}
\pgfpathcurveto{\pgfqpoint{1.331cm}{1.3cm}}{\pgfqpoint{1.345cm}{1.335cm}}{\pgfqpoint{1.345cm}{1.371cm}}
\pgfusepath{fill}
\begin{pgfscope}
\pgfsetdash{}{0cm}
\pgfsetlinewidth{0.818mm}
\pgfsetroundcap
\pgfsetmiterlimit{4.0}
\pgfpathmoveto{\pgfqpoint{0.682cm}{0.671cm}}
\pgfpathlineto{\pgfqpoint{0.682cm}{0.042cm}}
\pgfusepath{stroke}
\end{pgfscope}
\end{pgfscope}
\end{pgfscope}
\end{pgfscope}
\end{tikzpicture}}}
     )\\
     &  & - 6\lambda X_{M, \varepsilon} \circ ( X_{M, \varepsilon}^{\!\resizebox{0.6em}{!}{
\begin{tikzpicture}
\pgfpathmoveto{\pgfqpoint{0cm}{-0.035cm}}
\pgfpathlineto{\pgfqpoint{1.376cm}{-0.035cm}}
\pgfpathlineto{\pgfqpoint{1.376cm}{1.552cm}}
\pgfpathlineto{\pgfqpoint{0cm}{1.552cm}}
\pgfpathclose
\pgfusepath{clip}
\begin{pgfscope}
\begin{pgfscope}
\pgfpathmoveto{\pgfqpoint{0cm}{-0.035cm}}
\pgfpathlineto{\pgfqpoint{1.376cm}{-0.035cm}}
\pgfpathlineto{\pgfqpoint{1.376cm}{1.552cm}}
\pgfpathlineto{\pgfqpoint{0cm}{1.552cm}}
\pgfpathclose
\pgfusepath{clip}
\begin{pgfscope}
\begin{pgfscope}
\pgfsetdash{}{0cm}
\pgfsetlinewidth{0.818mm}
\pgfsetroundcap
\pgfsetroundjoin
\pgfsetmiterlimit{7.0}
\definecolor{eps2pgf_color}{gray}{0}\pgfsetstrokecolor{eps2pgf_color}\pgfsetfillcolor{eps2pgf_color}
\pgfpathmoveto{\pgfqpoint{0.117cm}{1.421cm}}
\pgfpathlineto{\pgfqpoint{0.682cm}{0.671cm}}
\pgfpathlineto{\pgfqpoint{1.246cm}{1.421cm}}
\pgfusepath{stroke}
\end{pgfscope}
\definecolor{eps2pgf_color}{gray}{0}\pgfsetstrokecolor{eps2pgf_color}\pgfsetfillcolor{eps2pgf_color}
\pgfpathmoveto{\pgfqpoint{0.273cm}{1.395cm}}
\pgfpathcurveto{\pgfqpoint{0.273cm}{1.432cm}}{\pgfqpoint{0.259cm}{1.467cm}}{\pgfqpoint{0.233cm}{1.492cm}}
\pgfpathcurveto{\pgfqpoint{0.207cm}{1.518cm}}{\pgfqpoint{0.173cm}{1.532cm}}{\pgfqpoint{0.137cm}{1.532cm}}
\pgfpathcurveto{\pgfqpoint{0.1cm}{1.532cm}}{\pgfqpoint{0.066cm}{1.518cm}}{\pgfqpoint{0.04cm}{1.492cm}}
\pgfpathcurveto{\pgfqpoint{0.014cm}{1.467cm}}{\pgfqpoint{0cm}{1.432cm}}{\pgfqpoint{0cm}{1.395cm}}
\pgfpathcurveto{\pgfqpoint{0cm}{1.359cm}}{\pgfqpoint{0.014cm}{1.324cm}}{\pgfqpoint{0.04cm}{1.299cm}}
\pgfpathcurveto{\pgfqpoint{0.066cm}{1.273cm}}{\pgfqpoint{0.1cm}{1.258cm}}{\pgfqpoint{0.137cm}{1.258cm}}
\pgfpathcurveto{\pgfqpoint{0.173cm}{1.258cm}}{\pgfqpoint{0.207cm}{1.273cm}}{\pgfqpoint{0.233cm}{1.299cm}}
\pgfpathcurveto{\pgfqpoint{0.259cm}{1.324cm}}{\pgfqpoint{0.273cm}{1.359cm}}{\pgfqpoint{0.273cm}{1.395cm}}
\pgfusepath{fill}
\begin{pgfscope}
\pgfsetdash{}{0cm}
\pgfsetlinewidth{0.818mm}
\pgfsetmiterlimit{7.0}
\pgfpathmoveto{\pgfqpoint{0.682cm}{0.671cm}}
\pgfpathlineto{\pgfqpoint{0.679cm}{1.418cm}}
\pgfusepath{stroke}
\end{pgfscope}
\pgfpathmoveto{\pgfqpoint{0.815cm}{1.399cm}}
\pgfpathcurveto{\pgfqpoint{0.815cm}{1.435cm}}{\pgfqpoint{0.801cm}{1.47cm}}{\pgfqpoint{0.775cm}{1.496cm}}
\pgfpathcurveto{\pgfqpoint{0.75cm}{1.521cm}}{\pgfqpoint{0.715cm}{1.536cm}}{\pgfqpoint{0.679cm}{1.536cm}}
\pgfpathcurveto{\pgfqpoint{0.643cm}{1.536cm}}{\pgfqpoint{0.608cm}{1.521cm}}{\pgfqpoint{0.582cm}{1.496cm}}
\pgfpathcurveto{\pgfqpoint{0.557cm}{1.47cm}}{\pgfqpoint{0.542cm}{1.435cm}}{\pgfqpoint{0.542cm}{1.399cm}}
\pgfpathcurveto{\pgfqpoint{0.542cm}{1.363cm}}{\pgfqpoint{0.557cm}{1.328cm}}{\pgfqpoint{0.582cm}{1.302cm}}
\pgfpathcurveto{\pgfqpoint{0.608cm}{1.276cm}}{\pgfqpoint{0.643cm}{1.262cm}}{\pgfqpoint{0.679cm}{1.262cm}}
\pgfpathcurveto{\pgfqpoint{0.715cm}{1.262cm}}{\pgfqpoint{0.75cm}{1.276cm}}{\pgfqpoint{0.775cm}{1.302cm}}
\pgfpathcurveto{\pgfqpoint{0.801cm}{1.328cm}}{\pgfqpoint{0.815cm}{1.363cm}}{\pgfqpoint{0.815cm}{1.399cm}}
\pgfusepath{fill}
\pgfpathmoveto{\pgfqpoint{1.345cm}{1.371cm}}
\pgfpathcurveto{\pgfqpoint{1.345cm}{1.408cm}}{\pgfqpoint{1.331cm}{1.442cm}}{\pgfqpoint{1.305cm}{1.468cm}}
\pgfpathcurveto{\pgfqpoint{1.28cm}{1.494cm}}{\pgfqpoint{1.245cm}{1.508cm}}{\pgfqpoint{1.209cm}{1.508cm}}
\pgfpathcurveto{\pgfqpoint{1.172cm}{1.508cm}}{\pgfqpoint{1.138cm}{1.494cm}}{\pgfqpoint{1.112cm}{1.468cm}}
\pgfpathcurveto{\pgfqpoint{1.087cm}{1.442cm}}{\pgfqpoint{1.072cm}{1.408cm}}{\pgfqpoint{1.072cm}{1.371cm}}
\pgfpathcurveto{\pgfqpoint{1.072cm}{1.335cm}}{\pgfqpoint{1.087cm}{1.3cm}}{\pgfqpoint{1.112cm}{1.274cm}}
\pgfpathcurveto{\pgfqpoint{1.138cm}{1.249cm}}{\pgfqpoint{1.172cm}{1.234cm}}{\pgfqpoint{1.209cm}{1.234cm}}
\pgfpathcurveto{\pgfqpoint{1.245cm}{1.234cm}}{\pgfqpoint{1.28cm}{1.249cm}}{\pgfqpoint{1.305cm}{1.274cm}}
\pgfpathcurveto{\pgfqpoint{1.331cm}{1.3cm}}{\pgfqpoint{1.345cm}{1.335cm}}{\pgfqpoint{1.345cm}{1.371cm}}
\pgfusepath{fill}
\begin{pgfscope}
\pgfsetdash{}{0cm}
\pgfsetlinewidth{0.818mm}
\pgfsetroundcap
\pgfsetmiterlimit{4.0}
\pgfpathmoveto{\pgfqpoint{0.682cm}{0.671cm}}
\pgfpathlineto{\pgfqpoint{0.682cm}{0.042cm}}
\pgfusepath{stroke}
\end{pgfscope}
\end{pgfscope}
\end{pgfscope}
\end{pgfscope}
\end{tikzpicture}}}
     \succ \zeta_{M, \varepsilon} ) - 6\lambda X_{M, \varepsilon} \circ (
     X_{M, \varepsilon}^{\!\resizebox{0.6em}{!}{
\begin{tikzpicture}
\pgfpathmoveto{\pgfqpoint{0cm}{-0.035cm}}
\pgfpathlineto{\pgfqpoint{1.376cm}{-0.035cm}}
\pgfpathlineto{\pgfqpoint{1.376cm}{1.552cm}}
\pgfpathlineto{\pgfqpoint{0cm}{1.552cm}}
\pgfpathclose
\pgfusepath{clip}
\begin{pgfscope}
\begin{pgfscope}
\pgfpathmoveto{\pgfqpoint{0cm}{-0.035cm}}
\pgfpathlineto{\pgfqpoint{1.376cm}{-0.035cm}}
\pgfpathlineto{\pgfqpoint{1.376cm}{1.552cm}}
\pgfpathlineto{\pgfqpoint{0cm}{1.552cm}}
\pgfpathclose
\pgfusepath{clip}
\begin{pgfscope}
\begin{pgfscope}
\pgfsetdash{}{0cm}
\pgfsetlinewidth{0.818mm}
\pgfsetroundcap
\pgfsetroundjoin
\pgfsetmiterlimit{7.0}
\definecolor{eps2pgf_color}{gray}{0}\pgfsetstrokecolor{eps2pgf_color}\pgfsetfillcolor{eps2pgf_color}
\pgfpathmoveto{\pgfqpoint{0.117cm}{1.421cm}}
\pgfpathlineto{\pgfqpoint{0.682cm}{0.671cm}}
\pgfpathlineto{\pgfqpoint{1.246cm}{1.421cm}}
\pgfusepath{stroke}
\end{pgfscope}
\definecolor{eps2pgf_color}{gray}{0}\pgfsetstrokecolor{eps2pgf_color}\pgfsetfillcolor{eps2pgf_color}
\pgfpathmoveto{\pgfqpoint{0.273cm}{1.395cm}}
\pgfpathcurveto{\pgfqpoint{0.273cm}{1.432cm}}{\pgfqpoint{0.259cm}{1.467cm}}{\pgfqpoint{0.233cm}{1.492cm}}
\pgfpathcurveto{\pgfqpoint{0.207cm}{1.518cm}}{\pgfqpoint{0.173cm}{1.532cm}}{\pgfqpoint{0.137cm}{1.532cm}}
\pgfpathcurveto{\pgfqpoint{0.1cm}{1.532cm}}{\pgfqpoint{0.066cm}{1.518cm}}{\pgfqpoint{0.04cm}{1.492cm}}
\pgfpathcurveto{\pgfqpoint{0.014cm}{1.467cm}}{\pgfqpoint{0cm}{1.432cm}}{\pgfqpoint{0cm}{1.395cm}}
\pgfpathcurveto{\pgfqpoint{0cm}{1.359cm}}{\pgfqpoint{0.014cm}{1.324cm}}{\pgfqpoint{0.04cm}{1.299cm}}
\pgfpathcurveto{\pgfqpoint{0.066cm}{1.273cm}}{\pgfqpoint{0.1cm}{1.258cm}}{\pgfqpoint{0.137cm}{1.258cm}}
\pgfpathcurveto{\pgfqpoint{0.173cm}{1.258cm}}{\pgfqpoint{0.207cm}{1.273cm}}{\pgfqpoint{0.233cm}{1.299cm}}
\pgfpathcurveto{\pgfqpoint{0.259cm}{1.324cm}}{\pgfqpoint{0.273cm}{1.359cm}}{\pgfqpoint{0.273cm}{1.395cm}}
\pgfusepath{fill}
\begin{pgfscope}
\pgfsetdash{}{0cm}
\pgfsetlinewidth{0.818mm}
\pgfsetmiterlimit{7.0}
\pgfpathmoveto{\pgfqpoint{0.682cm}{0.671cm}}
\pgfpathlineto{\pgfqpoint{0.679cm}{1.418cm}}
\pgfusepath{stroke}
\end{pgfscope}
\pgfpathmoveto{\pgfqpoint{0.815cm}{1.399cm}}
\pgfpathcurveto{\pgfqpoint{0.815cm}{1.435cm}}{\pgfqpoint{0.801cm}{1.47cm}}{\pgfqpoint{0.775cm}{1.496cm}}
\pgfpathcurveto{\pgfqpoint{0.75cm}{1.521cm}}{\pgfqpoint{0.715cm}{1.536cm}}{\pgfqpoint{0.679cm}{1.536cm}}
\pgfpathcurveto{\pgfqpoint{0.643cm}{1.536cm}}{\pgfqpoint{0.608cm}{1.521cm}}{\pgfqpoint{0.582cm}{1.496cm}}
\pgfpathcurveto{\pgfqpoint{0.557cm}{1.47cm}}{\pgfqpoint{0.542cm}{1.435cm}}{\pgfqpoint{0.542cm}{1.399cm}}
\pgfpathcurveto{\pgfqpoint{0.542cm}{1.363cm}}{\pgfqpoint{0.557cm}{1.328cm}}{\pgfqpoint{0.582cm}{1.302cm}}
\pgfpathcurveto{\pgfqpoint{0.608cm}{1.276cm}}{\pgfqpoint{0.643cm}{1.262cm}}{\pgfqpoint{0.679cm}{1.262cm}}
\pgfpathcurveto{\pgfqpoint{0.715cm}{1.262cm}}{\pgfqpoint{0.75cm}{1.276cm}}{\pgfqpoint{0.775cm}{1.302cm}}
\pgfpathcurveto{\pgfqpoint{0.801cm}{1.328cm}}{\pgfqpoint{0.815cm}{1.363cm}}{\pgfqpoint{0.815cm}{1.399cm}}
\pgfusepath{fill}
\pgfpathmoveto{\pgfqpoint{1.345cm}{1.371cm}}
\pgfpathcurveto{\pgfqpoint{1.345cm}{1.408cm}}{\pgfqpoint{1.331cm}{1.442cm}}{\pgfqpoint{1.305cm}{1.468cm}}
\pgfpathcurveto{\pgfqpoint{1.28cm}{1.494cm}}{\pgfqpoint{1.245cm}{1.508cm}}{\pgfqpoint{1.209cm}{1.508cm}}
\pgfpathcurveto{\pgfqpoint{1.172cm}{1.508cm}}{\pgfqpoint{1.138cm}{1.494cm}}{\pgfqpoint{1.112cm}{1.468cm}}
\pgfpathcurveto{\pgfqpoint{1.087cm}{1.442cm}}{\pgfqpoint{1.072cm}{1.408cm}}{\pgfqpoint{1.072cm}{1.371cm}}
\pgfpathcurveto{\pgfqpoint{1.072cm}{1.335cm}}{\pgfqpoint{1.087cm}{1.3cm}}{\pgfqpoint{1.112cm}{1.274cm}}
\pgfpathcurveto{\pgfqpoint{1.138cm}{1.249cm}}{\pgfqpoint{1.172cm}{1.234cm}}{\pgfqpoint{1.209cm}{1.234cm}}
\pgfpathcurveto{\pgfqpoint{1.245cm}{1.234cm}}{\pgfqpoint{1.28cm}{1.249cm}}{\pgfqpoint{1.305cm}{1.274cm}}
\pgfpathcurveto{\pgfqpoint{1.331cm}{1.3cm}}{\pgfqpoint{1.345cm}{1.335cm}}{\pgfqpoint{1.345cm}{1.371cm}}
\pgfusepath{fill}
\begin{pgfscope}
\pgfsetdash{}{0cm}
\pgfsetlinewidth{0.818mm}
\pgfsetroundcap
\pgfsetmiterlimit{4.0}
\pgfpathmoveto{\pgfqpoint{0.682cm}{0.671cm}}
\pgfpathlineto{\pgfqpoint{0.682cm}{0.042cm}}
\pgfusepath{stroke}
\end{pgfscope}
\end{pgfscope}
\end{pgfscope}
\end{pgfscope}
\end{tikzpicture}}} \preccurlyeq \zeta_{M, \varepsilon}
     ) \\
     & &+ 3 X_{M, \varepsilon} \circ \zeta_{M, \varepsilon}^2\\
     & = & 6\lambda (\lambdaX_{M, \varepsilon}^{\!\resizebox{0.6em}{!}{
\begin{tikzpicture}
\pgfpathmoveto{\pgfqpoint{0cm}{-0.035cm}}
\pgfpathlineto{\pgfqpoint{1.376cm}{-0.035cm}}
\pgfpathlineto{\pgfqpoint{1.376cm}{1.552cm}}
\pgfpathlineto{\pgfqpoint{0cm}{1.552cm}}
\pgfpathclose
\pgfusepath{clip}
\begin{pgfscope}
\begin{pgfscope}
\pgfpathmoveto{\pgfqpoint{0cm}{-0.035cm}}
\pgfpathlineto{\pgfqpoint{1.376cm}{-0.035cm}}
\pgfpathlineto{\pgfqpoint{1.376cm}{1.552cm}}
\pgfpathlineto{\pgfqpoint{0cm}{1.552cm}}
\pgfpathclose
\pgfusepath{clip}
\begin{pgfscope}
\begin{pgfscope}
\pgfsetdash{}{0cm}
\pgfsetlinewidth{0.818mm}
\pgfsetroundcap
\pgfsetroundjoin
\pgfsetmiterlimit{7.0}
\definecolor{eps2pgf_color}{gray}{0}\pgfsetstrokecolor{eps2pgf_color}\pgfsetfillcolor{eps2pgf_color}
\pgfpathmoveto{\pgfqpoint{0.117cm}{1.421cm}}
\pgfpathlineto{\pgfqpoint{0.682cm}{0.671cm}}
\pgfpathlineto{\pgfqpoint{1.246cm}{1.421cm}}
\pgfusepath{stroke}
\end{pgfscope}
\definecolor{eps2pgf_color}{gray}{0}\pgfsetstrokecolor{eps2pgf_color}\pgfsetfillcolor{eps2pgf_color}
\pgfpathmoveto{\pgfqpoint{0.273cm}{1.395cm}}
\pgfpathcurveto{\pgfqpoint{0.273cm}{1.432cm}}{\pgfqpoint{0.259cm}{1.467cm}}{\pgfqpoint{0.233cm}{1.492cm}}
\pgfpathcurveto{\pgfqpoint{0.207cm}{1.518cm}}{\pgfqpoint{0.173cm}{1.532cm}}{\pgfqpoint{0.137cm}{1.532cm}}
\pgfpathcurveto{\pgfqpoint{0.1cm}{1.532cm}}{\pgfqpoint{0.066cm}{1.518cm}}{\pgfqpoint{0.04cm}{1.492cm}}
\pgfpathcurveto{\pgfqpoint{0.014cm}{1.467cm}}{\pgfqpoint{0cm}{1.432cm}}{\pgfqpoint{0cm}{1.395cm}}
\pgfpathcurveto{\pgfqpoint{0cm}{1.359cm}}{\pgfqpoint{0.014cm}{1.324cm}}{\pgfqpoint{0.04cm}{1.299cm}}
\pgfpathcurveto{\pgfqpoint{0.066cm}{1.273cm}}{\pgfqpoint{0.1cm}{1.258cm}}{\pgfqpoint{0.137cm}{1.258cm}}
\pgfpathcurveto{\pgfqpoint{0.173cm}{1.258cm}}{\pgfqpoint{0.207cm}{1.273cm}}{\pgfqpoint{0.233cm}{1.299cm}}
\pgfpathcurveto{\pgfqpoint{0.259cm}{1.324cm}}{\pgfqpoint{0.273cm}{1.359cm}}{\pgfqpoint{0.273cm}{1.395cm}}
\pgfusepath{fill}
\begin{pgfscope}
\pgfsetdash{}{0cm}
\pgfsetlinewidth{0.818mm}
\pgfsetmiterlimit{7.0}
\pgfpathmoveto{\pgfqpoint{0.682cm}{0.671cm}}
\pgfpathlineto{\pgfqpoint{0.679cm}{1.418cm}}
\pgfusepath{stroke}
\end{pgfscope}
\pgfpathmoveto{\pgfqpoint{0.815cm}{1.399cm}}
\pgfpathcurveto{\pgfqpoint{0.815cm}{1.435cm}}{\pgfqpoint{0.801cm}{1.47cm}}{\pgfqpoint{0.775cm}{1.496cm}}
\pgfpathcurveto{\pgfqpoint{0.75cm}{1.521cm}}{\pgfqpoint{0.715cm}{1.536cm}}{\pgfqpoint{0.679cm}{1.536cm}}
\pgfpathcurveto{\pgfqpoint{0.643cm}{1.536cm}}{\pgfqpoint{0.608cm}{1.521cm}}{\pgfqpoint{0.582cm}{1.496cm}}
\pgfpathcurveto{\pgfqpoint{0.557cm}{1.47cm}}{\pgfqpoint{0.542cm}{1.435cm}}{\pgfqpoint{0.542cm}{1.399cm}}
\pgfpathcurveto{\pgfqpoint{0.542cm}{1.363cm}}{\pgfqpoint{0.557cm}{1.328cm}}{\pgfqpoint{0.582cm}{1.302cm}}
\pgfpathcurveto{\pgfqpoint{0.608cm}{1.276cm}}{\pgfqpoint{0.643cm}{1.262cm}}{\pgfqpoint{0.679cm}{1.262cm}}
\pgfpathcurveto{\pgfqpoint{0.715cm}{1.262cm}}{\pgfqpoint{0.75cm}{1.276cm}}{\pgfqpoint{0.775cm}{1.302cm}}
\pgfpathcurveto{\pgfqpoint{0.801cm}{1.328cm}}{\pgfqpoint{0.815cm}{1.363cm}}{\pgfqpoint{0.815cm}{1.399cm}}
\pgfusepath{fill}
\pgfpathmoveto{\pgfqpoint{1.345cm}{1.371cm}}
\pgfpathcurveto{\pgfqpoint{1.345cm}{1.408cm}}{\pgfqpoint{1.331cm}{1.442cm}}{\pgfqpoint{1.305cm}{1.468cm}}
\pgfpathcurveto{\pgfqpoint{1.28cm}{1.494cm}}{\pgfqpoint{1.245cm}{1.508cm}}{\pgfqpoint{1.209cm}{1.508cm}}
\pgfpathcurveto{\pgfqpoint{1.172cm}{1.508cm}}{\pgfqpoint{1.138cm}{1.494cm}}{\pgfqpoint{1.112cm}{1.468cm}}
\pgfpathcurveto{\pgfqpoint{1.087cm}{1.442cm}}{\pgfqpoint{1.072cm}{1.408cm}}{\pgfqpoint{1.072cm}{1.371cm}}
\pgfpathcurveto{\pgfqpoint{1.072cm}{1.335cm}}{\pgfqpoint{1.087cm}{1.3cm}}{\pgfqpoint{1.112cm}{1.274cm}}
\pgfpathcurveto{\pgfqpoint{1.138cm}{1.249cm}}{\pgfqpoint{1.172cm}{1.234cm}}{\pgfqpoint{1.209cm}{1.234cm}}
\pgfpathcurveto{\pgfqpoint{1.245cm}{1.234cm}}{\pgfqpoint{1.28cm}{1.249cm}}{\pgfqpoint{1.305cm}{1.274cm}}
\pgfpathcurveto{\pgfqpoint{1.331cm}{1.3cm}}{\pgfqpoint{1.345cm}{1.335cm}}{\pgfqpoint{1.345cm}{1.371cm}}
\pgfusepath{fill}
\begin{pgfscope}
\pgfsetdash{}{0cm}
\pgfsetlinewidth{0.818mm}
\pgfsetroundcap
\pgfsetmiterlimit{4.0}
\pgfpathmoveto{\pgfqpoint{0.682cm}{0.671cm}}
\pgfpathlineto{\pgfqpoint{0.682cm}{0.042cm}}
\pgfusepath{stroke}
\end{pgfscope}
\end{pgfscope}
\end{pgfscope}
\end{pgfscope}
\end{tikzpicture}}} - \zeta_{M, \varepsilon})
     X_{M, \varepsilon}^{\!\resizebox{!}{.8em}{
\begin{tikzpicture}
\pgfpathmoveto{\pgfqpoint{0cm}{-0.035cm}}
\pgfpathlineto{\pgfqpoint{1.976cm}{-0.035cm}}
\pgfpathlineto{\pgfqpoint{1.976cm}{1.94cm}}
\pgfpathlineto{\pgfqpoint{0cm}{1.94cm}}
\pgfpathclose
\pgfusepath{clip}
\begin{pgfscope}
\begin{pgfscope}
\pgfpathmoveto{\pgfqpoint{0cm}{-0.035cm}}
\pgfpathlineto{\pgfqpoint{1.976cm}{-0.035cm}}
\pgfpathlineto{\pgfqpoint{1.976cm}{1.94cm}}
\pgfpathlineto{\pgfqpoint{0cm}{1.94cm}}
\pgfpathclose
\pgfusepath{clip}
\begin{pgfscope}
\begin{pgfscope}
\pgfsetdash{}{0cm}
\pgfsetlinewidth{0.818mm}
\pgfsetroundcap
\pgfsetroundjoin
\pgfsetmiterlimit{7.0}
\definecolor{eps2pgf_color}{gray}{0}\pgfsetstrokecolor{eps2pgf_color}\pgfsetfillcolor{eps2pgf_color}
\pgfpathmoveto{\pgfqpoint{0.117cm}{1.815cm}}
\pgfpathlineto{\pgfqpoint{0.682cm}{1.065cm}}
\pgfpathlineto{\pgfqpoint{1.246cm}{1.815cm}}
\pgfusepath{stroke}
\end{pgfscope}
\definecolor{eps2pgf_color}{gray}{0}\pgfsetstrokecolor{eps2pgf_color}\pgfsetfillcolor{eps2pgf_color}
\pgfpathmoveto{\pgfqpoint{0.273cm}{1.789cm}}
\pgfpathcurveto{\pgfqpoint{0.273cm}{1.825cm}}{\pgfqpoint{0.259cm}{1.86cm}}{\pgfqpoint{0.233cm}{1.886cm}}
\pgfpathcurveto{\pgfqpoint{0.207cm}{1.912cm}}{\pgfqpoint{0.173cm}{1.926cm}}{\pgfqpoint{0.137cm}{1.926cm}}
\pgfpathcurveto{\pgfqpoint{0.1cm}{1.926cm}}{\pgfqpoint{0.066cm}{1.912cm}}{\pgfqpoint{0.04cm}{1.886cm}}
\pgfpathcurveto{\pgfqpoint{0.014cm}{1.86cm}}{\pgfqpoint{0cm}{1.825cm}}{\pgfqpoint{0cm}{1.789cm}}
\pgfpathcurveto{\pgfqpoint{0cm}{1.753cm}}{\pgfqpoint{0.014cm}{1.718cm}}{\pgfqpoint{0.04cm}{1.692cm}}
\pgfpathcurveto{\pgfqpoint{0.066cm}{1.667cm}}{\pgfqpoint{0.1cm}{1.652cm}}{\pgfqpoint{0.137cm}{1.652cm}}
\pgfpathcurveto{\pgfqpoint{0.173cm}{1.652cm}}{\pgfqpoint{0.207cm}{1.667cm}}{\pgfqpoint{0.233cm}{1.692cm}}
\pgfpathcurveto{\pgfqpoint{0.259cm}{1.718cm}}{\pgfqpoint{0.273cm}{1.753cm}}{\pgfqpoint{0.273cm}{1.789cm}}
\pgfusepath{fill}
\begin{pgfscope}
\pgfsetdash{}{0cm}
\pgfsetlinewidth{0.818mm}
\pgfsetmiterlimit{7.0}
\pgfpathmoveto{\pgfqpoint{0.682cm}{1.065cm}}
\pgfpathlineto{\pgfqpoint{0.679cm}{1.812cm}}
\pgfusepath{stroke}
\end{pgfscope}
\pgfpathmoveto{\pgfqpoint{0.815cm}{1.793cm}}
\pgfpathcurveto{\pgfqpoint{0.815cm}{1.829cm}}{\pgfqpoint{0.801cm}{1.864cm}}{\pgfqpoint{0.775cm}{1.89cm}}
\pgfpathcurveto{\pgfqpoint{0.75cm}{1.915cm}}{\pgfqpoint{0.715cm}{1.93cm}}{\pgfqpoint{0.679cm}{1.93cm}}
\pgfpathcurveto{\pgfqpoint{0.643cm}{1.93cm}}{\pgfqpoint{0.608cm}{1.915cm}}{\pgfqpoint{0.582cm}{1.89cm}}
\pgfpathcurveto{\pgfqpoint{0.557cm}{1.864cm}}{\pgfqpoint{0.542cm}{1.829cm}}{\pgfqpoint{0.542cm}{1.793cm}}
\pgfpathcurveto{\pgfqpoint{0.542cm}{1.756cm}}{\pgfqpoint{0.557cm}{1.722cm}}{\pgfqpoint{0.582cm}{1.696cm}}
\pgfpathcurveto{\pgfqpoint{0.608cm}{1.67cm}}{\pgfqpoint{0.643cm}{1.656cm}}{\pgfqpoint{0.679cm}{1.656cm}}
\pgfpathcurveto{\pgfqpoint{0.715cm}{1.656cm}}{\pgfqpoint{0.75cm}{1.67cm}}{\pgfqpoint{0.775cm}{1.696cm}}
\pgfpathcurveto{\pgfqpoint{0.801cm}{1.722cm}}{\pgfqpoint{0.815cm}{1.756cm}}{\pgfqpoint{0.815cm}{1.793cm}}
\pgfusepath{fill}
\pgfpathmoveto{\pgfqpoint{1.345cm}{1.765cm}}
\pgfpathcurveto{\pgfqpoint{1.345cm}{1.801cm}}{\pgfqpoint{1.331cm}{1.836cm}}{\pgfqpoint{1.305cm}{1.862cm}}
\pgfpathcurveto{\pgfqpoint{1.28cm}{1.887cm}}{\pgfqpoint{1.245cm}{1.902cm}}{\pgfqpoint{1.209cm}{1.902cm}}
\pgfpathcurveto{\pgfqpoint{1.172cm}{1.902cm}}{\pgfqpoint{1.138cm}{1.887cm}}{\pgfqpoint{1.112cm}{1.862cm}}
\pgfpathcurveto{\pgfqpoint{1.087cm}{1.836cm}}{\pgfqpoint{1.072cm}{1.801cm}}{\pgfqpoint{1.072cm}{1.765cm}}
\pgfpathcurveto{\pgfqpoint{1.072cm}{1.728cm}}{\pgfqpoint{1.087cm}{1.694cm}}{\pgfqpoint{1.112cm}{1.668cm}}
\pgfpathcurveto{\pgfqpoint{1.138cm}{1.642cm}}{\pgfqpoint{1.172cm}{1.628cm}}{\pgfqpoint{1.209cm}{1.628cm}}
\pgfpathcurveto{\pgfqpoint{1.245cm}{1.628cm}}{\pgfqpoint{1.28cm}{1.642cm}}{\pgfqpoint{1.305cm}{1.668cm}}
\pgfpathcurveto{\pgfqpoint{1.331cm}{1.694cm}}{\pgfqpoint{1.345cm}{1.728cm}}{\pgfqpoint{1.345cm}{1.765cm}}
\pgfusepath{fill}
\begin{pgfscope}
\pgfsetdash{}{0cm}
\pgfsetlinewidth{0.818mm}
\pgfsetroundcap
\pgfsetroundjoin
\pgfsetmiterlimit{7.0}
\pgfpathmoveto{\pgfqpoint{0.682cm}{1.065cm}}
\pgfpathlineto{\pgfqpoint{1.246cm}{0.315cm}}
\pgfpathlineto{\pgfqpoint{1.811cm}{1.065cm}}
\pgfusepath{stroke}
\end{pgfscope}
\pgfpathmoveto{\pgfqpoint{1.948cm}{1.065cm}}
\pgfpathcurveto{\pgfqpoint{1.948cm}{1.101cm}}{\pgfqpoint{1.933cm}{1.136cm}}{\pgfqpoint{1.907cm}{1.162cm}}
\pgfpathcurveto{\pgfqpoint{1.882cm}{1.187cm}}{\pgfqpoint{1.847cm}{1.202cm}}{\pgfqpoint{1.811cm}{1.202cm}}
\pgfpathcurveto{\pgfqpoint{1.775cm}{1.202cm}}{\pgfqpoint{1.74cm}{1.187cm}}{\pgfqpoint{1.714cm}{1.162cm}}
\pgfpathcurveto{\pgfqpoint{1.689cm}{1.136cm}}{\pgfqpoint{1.674cm}{1.101cm}}{\pgfqpoint{1.674cm}{1.065cm}}
\pgfpathcurveto{\pgfqpoint{1.674cm}{1.029cm}}{\pgfqpoint{1.689cm}{0.994cm}}{\pgfqpoint{1.714cm}{0.968cm}}
\pgfpathcurveto{\pgfqpoint{1.74cm}{0.942cm}}{\pgfqpoint{1.775cm}{0.928cm}}{\pgfqpoint{1.811cm}{0.928cm}}
\pgfpathcurveto{\pgfqpoint{1.847cm}{0.928cm}}{\pgfqpoint{1.882cm}{0.942cm}}{\pgfqpoint{1.907cm}{0.968cm}}
\pgfpathcurveto{\pgfqpoint{1.933cm}{0.994cm}}{\pgfqpoint{1.948cm}{1.029cm}}{\pgfqpoint{1.948cm}{1.065cm}}
\pgfusepath{fill}
\begin{pgfscope}
\pgfsetdash{}{0cm}
\pgfsetlinewidth{0.818mm}
\pgfsetmiterlimit{4.0}
\pgfpathmoveto{\pgfqpoint{1.383cm}{0.178cm}}
\pgfpathcurveto{\pgfqpoint{1.383cm}{0.214cm}}{\pgfqpoint{1.369cm}{0.249cm}}{\pgfqpoint{1.343cm}{0.275cm}}
\pgfpathcurveto{\pgfqpoint{1.317cm}{0.3cm}}{\pgfqpoint{1.283cm}{0.315cm}}{\pgfqpoint{1.246cm}{0.315cm}}
\pgfpathcurveto{\pgfqpoint{1.21cm}{0.315cm}}{\pgfqpoint{1.175cm}{0.3cm}}{\pgfqpoint{1.15cm}{0.275cm}}
\pgfpathcurveto{\pgfqpoint{1.124cm}{0.249cm}}{\pgfqpoint{1.11cm}{0.214cm}}{\pgfqpoint{1.11cm}{0.178cm}}
\pgfpathcurveto{\pgfqpoint{1.11cm}{0.141cm}}{\pgfqpoint{1.124cm}{0.107cm}}{\pgfqpoint{1.15cm}{0.081cm}}
\pgfpathcurveto{\pgfqpoint{1.175cm}{0.055cm}}{\pgfqpoint{1.21cm}{0.041cm}}{\pgfqpoint{1.246cm}{0.041cm}}
\pgfpathcurveto{\pgfqpoint{1.283cm}{0.041cm}}{\pgfqpoint{1.317cm}{0.055cm}}{\pgfqpoint{1.343cm}{0.081cm}}
\pgfpathcurveto{\pgfqpoint{1.369cm}{0.107cm}}{\pgfqpoint{1.383cm}{0.141cm}}{\pgfqpoint{1.383cm}{0.178cm}}
\pgfusepath{stroke}
\end{pgfscope}
\end{pgfscope}
\end{pgfscope}
\end{pgfscope}
\end{tikzpicture}}} + 6\lambda C_{\varepsilon} (\lambda
     X_{M, \varepsilon}^{\!\resizebox{0.6em}{!}{
\begin{tikzpicture}
\pgfpathmoveto{\pgfqpoint{0cm}{-0.035cm}}
\pgfpathlineto{\pgfqpoint{1.376cm}{-0.035cm}}
\pgfpathlineto{\pgfqpoint{1.376cm}{1.552cm}}
\pgfpathlineto{\pgfqpoint{0cm}{1.552cm}}
\pgfpathclose
\pgfusepath{clip}
\begin{pgfscope}
\begin{pgfscope}
\pgfpathmoveto{\pgfqpoint{0cm}{-0.035cm}}
\pgfpathlineto{\pgfqpoint{1.376cm}{-0.035cm}}
\pgfpathlineto{\pgfqpoint{1.376cm}{1.552cm}}
\pgfpathlineto{\pgfqpoint{0cm}{1.552cm}}
\pgfpathclose
\pgfusepath{clip}
\begin{pgfscope}
\begin{pgfscope}
\pgfsetdash{}{0cm}
\pgfsetlinewidth{0.818mm}
\pgfsetroundcap
\pgfsetroundjoin
\pgfsetmiterlimit{7.0}
\definecolor{eps2pgf_color}{gray}{0}\pgfsetstrokecolor{eps2pgf_color}\pgfsetfillcolor{eps2pgf_color}
\pgfpathmoveto{\pgfqpoint{0.117cm}{1.421cm}}
\pgfpathlineto{\pgfqpoint{0.682cm}{0.671cm}}
\pgfpathlineto{\pgfqpoint{1.246cm}{1.421cm}}
\pgfusepath{stroke}
\end{pgfscope}
\definecolor{eps2pgf_color}{gray}{0}\pgfsetstrokecolor{eps2pgf_color}\pgfsetfillcolor{eps2pgf_color}
\pgfpathmoveto{\pgfqpoint{0.273cm}{1.395cm}}
\pgfpathcurveto{\pgfqpoint{0.273cm}{1.432cm}}{\pgfqpoint{0.259cm}{1.467cm}}{\pgfqpoint{0.233cm}{1.492cm}}
\pgfpathcurveto{\pgfqpoint{0.207cm}{1.518cm}}{\pgfqpoint{0.173cm}{1.532cm}}{\pgfqpoint{0.137cm}{1.532cm}}
\pgfpathcurveto{\pgfqpoint{0.1cm}{1.532cm}}{\pgfqpoint{0.066cm}{1.518cm}}{\pgfqpoint{0.04cm}{1.492cm}}
\pgfpathcurveto{\pgfqpoint{0.014cm}{1.467cm}}{\pgfqpoint{0cm}{1.432cm}}{\pgfqpoint{0cm}{1.395cm}}
\pgfpathcurveto{\pgfqpoint{0cm}{1.359cm}}{\pgfqpoint{0.014cm}{1.324cm}}{\pgfqpoint{0.04cm}{1.299cm}}
\pgfpathcurveto{\pgfqpoint{0.066cm}{1.273cm}}{\pgfqpoint{0.1cm}{1.258cm}}{\pgfqpoint{0.137cm}{1.258cm}}
\pgfpathcurveto{\pgfqpoint{0.173cm}{1.258cm}}{\pgfqpoint{0.207cm}{1.273cm}}{\pgfqpoint{0.233cm}{1.299cm}}
\pgfpathcurveto{\pgfqpoint{0.259cm}{1.324cm}}{\pgfqpoint{0.273cm}{1.359cm}}{\pgfqpoint{0.273cm}{1.395cm}}
\pgfusepath{fill}
\begin{pgfscope}
\pgfsetdash{}{0cm}
\pgfsetlinewidth{0.818mm}
\pgfsetmiterlimit{7.0}
\pgfpathmoveto{\pgfqpoint{0.682cm}{0.671cm}}
\pgfpathlineto{\pgfqpoint{0.679cm}{1.418cm}}
\pgfusepath{stroke}
\end{pgfscope}
\pgfpathmoveto{\pgfqpoint{0.815cm}{1.399cm}}
\pgfpathcurveto{\pgfqpoint{0.815cm}{1.435cm}}{\pgfqpoint{0.801cm}{1.47cm}}{\pgfqpoint{0.775cm}{1.496cm}}
\pgfpathcurveto{\pgfqpoint{0.75cm}{1.521cm}}{\pgfqpoint{0.715cm}{1.536cm}}{\pgfqpoint{0.679cm}{1.536cm}}
\pgfpathcurveto{\pgfqpoint{0.643cm}{1.536cm}}{\pgfqpoint{0.608cm}{1.521cm}}{\pgfqpoint{0.582cm}{1.496cm}}
\pgfpathcurveto{\pgfqpoint{0.557cm}{1.47cm}}{\pgfqpoint{0.542cm}{1.435cm}}{\pgfqpoint{0.542cm}{1.399cm}}
\pgfpathcurveto{\pgfqpoint{0.542cm}{1.363cm}}{\pgfqpoint{0.557cm}{1.328cm}}{\pgfqpoint{0.582cm}{1.302cm}}
\pgfpathcurveto{\pgfqpoint{0.608cm}{1.276cm}}{\pgfqpoint{0.643cm}{1.262cm}}{\pgfqpoint{0.679cm}{1.262cm}}
\pgfpathcurveto{\pgfqpoint{0.715cm}{1.262cm}}{\pgfqpoint{0.75cm}{1.276cm}}{\pgfqpoint{0.775cm}{1.302cm}}
\pgfpathcurveto{\pgfqpoint{0.801cm}{1.328cm}}{\pgfqpoint{0.815cm}{1.363cm}}{\pgfqpoint{0.815cm}{1.399cm}}
\pgfusepath{fill}
\pgfpathmoveto{\pgfqpoint{1.345cm}{1.371cm}}
\pgfpathcurveto{\pgfqpoint{1.345cm}{1.408cm}}{\pgfqpoint{1.331cm}{1.442cm}}{\pgfqpoint{1.305cm}{1.468cm}}
\pgfpathcurveto{\pgfqpoint{1.28cm}{1.494cm}}{\pgfqpoint{1.245cm}{1.508cm}}{\pgfqpoint{1.209cm}{1.508cm}}
\pgfpathcurveto{\pgfqpoint{1.172cm}{1.508cm}}{\pgfqpoint{1.138cm}{1.494cm}}{\pgfqpoint{1.112cm}{1.468cm}}
\pgfpathcurveto{\pgfqpoint{1.087cm}{1.442cm}}{\pgfqpoint{1.072cm}{1.408cm}}{\pgfqpoint{1.072cm}{1.371cm}}
\pgfpathcurveto{\pgfqpoint{1.072cm}{1.335cm}}{\pgfqpoint{1.087cm}{1.3cm}}{\pgfqpoint{1.112cm}{1.274cm}}
\pgfpathcurveto{\pgfqpoint{1.138cm}{1.249cm}}{\pgfqpoint{1.172cm}{1.234cm}}{\pgfqpoint{1.209cm}{1.234cm}}
\pgfpathcurveto{\pgfqpoint{1.245cm}{1.234cm}}{\pgfqpoint{1.28cm}{1.249cm}}{\pgfqpoint{1.305cm}{1.274cm}}
\pgfpathcurveto{\pgfqpoint{1.331cm}{1.3cm}}{\pgfqpoint{1.345cm}{1.335cm}}{\pgfqpoint{1.345cm}{1.371cm}}
\pgfusepath{fill}
\begin{pgfscope}
\pgfsetdash{}{0cm}
\pgfsetlinewidth{0.818mm}
\pgfsetroundcap
\pgfsetmiterlimit{4.0}
\pgfpathmoveto{\pgfqpoint{0.682cm}{0.671cm}}
\pgfpathlineto{\pgfqpoint{0.682cm}{0.042cm}}
\pgfusepath{stroke}
\end{pgfscope}
\end{pgfscope}
\end{pgfscope}
\end{pgfscope}
\end{tikzpicture}}} - \zeta_{M, \varepsilon}, X_{M,
     \varepsilon}^{\!\resizebox{0.6em}{!}{
\begin{tikzpicture}
\pgfpathmoveto{\pgfqpoint{0cm}{-0.035cm}}
\pgfpathlineto{\pgfqpoint{1.376cm}{-0.035cm}}
\pgfpathlineto{\pgfqpoint{1.376cm}{1.552cm}}
\pgfpathlineto{\pgfqpoint{0cm}{1.552cm}}
\pgfpathclose
\pgfusepath{clip}
\begin{pgfscope}
\begin{pgfscope}
\pgfpathmoveto{\pgfqpoint{0cm}{-0.035cm}}
\pgfpathlineto{\pgfqpoint{1.376cm}{-0.035cm}}
\pgfpathlineto{\pgfqpoint{1.376cm}{1.552cm}}
\pgfpathlineto{\pgfqpoint{0cm}{1.552cm}}
\pgfpathclose
\pgfusepath{clip}
\begin{pgfscope}
\begin{pgfscope}
\pgfsetdash{}{0cm}
\pgfsetlinewidth{0.818mm}
\pgfsetroundcap
\pgfsetroundjoin
\pgfsetmiterlimit{7.0}
\definecolor{eps2pgf_color}{gray}{0}\pgfsetstrokecolor{eps2pgf_color}\pgfsetfillcolor{eps2pgf_color}
\pgfpathmoveto{\pgfqpoint{0.117cm}{1.421cm}}
\pgfpathlineto{\pgfqpoint{0.682cm}{0.671cm}}
\pgfpathlineto{\pgfqpoint{1.246cm}{1.421cm}}
\pgfusepath{stroke}
\end{pgfscope}
\definecolor{eps2pgf_color}{gray}{0}\pgfsetstrokecolor{eps2pgf_color}\pgfsetfillcolor{eps2pgf_color}
\pgfpathmoveto{\pgfqpoint{0.273cm}{1.395cm}}
\pgfpathcurveto{\pgfqpoint{0.273cm}{1.432cm}}{\pgfqpoint{0.259cm}{1.467cm}}{\pgfqpoint{0.233cm}{1.492cm}}
\pgfpathcurveto{\pgfqpoint{0.207cm}{1.518cm}}{\pgfqpoint{0.173cm}{1.532cm}}{\pgfqpoint{0.137cm}{1.532cm}}
\pgfpathcurveto{\pgfqpoint{0.1cm}{1.532cm}}{\pgfqpoint{0.066cm}{1.518cm}}{\pgfqpoint{0.04cm}{1.492cm}}
\pgfpathcurveto{\pgfqpoint{0.014cm}{1.467cm}}{\pgfqpoint{0cm}{1.432cm}}{\pgfqpoint{0cm}{1.395cm}}
\pgfpathcurveto{\pgfqpoint{0cm}{1.359cm}}{\pgfqpoint{0.014cm}{1.324cm}}{\pgfqpoint{0.04cm}{1.299cm}}
\pgfpathcurveto{\pgfqpoint{0.066cm}{1.273cm}}{\pgfqpoint{0.1cm}{1.258cm}}{\pgfqpoint{0.137cm}{1.258cm}}
\pgfpathcurveto{\pgfqpoint{0.173cm}{1.258cm}}{\pgfqpoint{0.207cm}{1.273cm}}{\pgfqpoint{0.233cm}{1.299cm}}
\pgfpathcurveto{\pgfqpoint{0.259cm}{1.324cm}}{\pgfqpoint{0.273cm}{1.359cm}}{\pgfqpoint{0.273cm}{1.395cm}}
\pgfusepath{fill}
\begin{pgfscope}
\pgfsetdash{}{0cm}
\pgfsetlinewidth{0.818mm}
\pgfsetmiterlimit{7.0}
\pgfpathmoveto{\pgfqpoint{0.682cm}{0.671cm}}
\pgfpathlineto{\pgfqpoint{0.679cm}{1.418cm}}
\pgfusepath{stroke}
\end{pgfscope}
\pgfpathmoveto{\pgfqpoint{0.815cm}{1.399cm}}
\pgfpathcurveto{\pgfqpoint{0.815cm}{1.435cm}}{\pgfqpoint{0.801cm}{1.47cm}}{\pgfqpoint{0.775cm}{1.496cm}}
\pgfpathcurveto{\pgfqpoint{0.75cm}{1.521cm}}{\pgfqpoint{0.715cm}{1.536cm}}{\pgfqpoint{0.679cm}{1.536cm}}
\pgfpathcurveto{\pgfqpoint{0.643cm}{1.536cm}}{\pgfqpoint{0.608cm}{1.521cm}}{\pgfqpoint{0.582cm}{1.496cm}}
\pgfpathcurveto{\pgfqpoint{0.557cm}{1.47cm}}{\pgfqpoint{0.542cm}{1.435cm}}{\pgfqpoint{0.542cm}{1.399cm}}
\pgfpathcurveto{\pgfqpoint{0.542cm}{1.363cm}}{\pgfqpoint{0.557cm}{1.328cm}}{\pgfqpoint{0.582cm}{1.302cm}}
\pgfpathcurveto{\pgfqpoint{0.608cm}{1.276cm}}{\pgfqpoint{0.643cm}{1.262cm}}{\pgfqpoint{0.679cm}{1.262cm}}
\pgfpathcurveto{\pgfqpoint{0.715cm}{1.262cm}}{\pgfqpoint{0.75cm}{1.276cm}}{\pgfqpoint{0.775cm}{1.302cm}}
\pgfpathcurveto{\pgfqpoint{0.801cm}{1.328cm}}{\pgfqpoint{0.815cm}{1.363cm}}{\pgfqpoint{0.815cm}{1.399cm}}
\pgfusepath{fill}
\pgfpathmoveto{\pgfqpoint{1.345cm}{1.371cm}}
\pgfpathcurveto{\pgfqpoint{1.345cm}{1.408cm}}{\pgfqpoint{1.331cm}{1.442cm}}{\pgfqpoint{1.305cm}{1.468cm}}
\pgfpathcurveto{\pgfqpoint{1.28cm}{1.494cm}}{\pgfqpoint{1.245cm}{1.508cm}}{\pgfqpoint{1.209cm}{1.508cm}}
\pgfpathcurveto{\pgfqpoint{1.172cm}{1.508cm}}{\pgfqpoint{1.138cm}{1.494cm}}{\pgfqpoint{1.112cm}{1.468cm}}
\pgfpathcurveto{\pgfqpoint{1.087cm}{1.442cm}}{\pgfqpoint{1.072cm}{1.408cm}}{\pgfqpoint{1.072cm}{1.371cm}}
\pgfpathcurveto{\pgfqpoint{1.072cm}{1.335cm}}{\pgfqpoint{1.087cm}{1.3cm}}{\pgfqpoint{1.112cm}{1.274cm}}
\pgfpathcurveto{\pgfqpoint{1.138cm}{1.249cm}}{\pgfqpoint{1.172cm}{1.234cm}}{\pgfqpoint{1.209cm}{1.234cm}}
\pgfpathcurveto{\pgfqpoint{1.245cm}{1.234cm}}{\pgfqpoint{1.28cm}{1.249cm}}{\pgfqpoint{1.305cm}{1.274cm}}
\pgfpathcurveto{\pgfqpoint{1.331cm}{1.3cm}}{\pgfqpoint{1.345cm}{1.335cm}}{\pgfqpoint{1.345cm}{1.371cm}}
\pgfusepath{fill}
\begin{pgfscope}
\pgfsetdash{}{0cm}
\pgfsetlinewidth{0.818mm}
\pgfsetroundcap
\pgfsetmiterlimit{4.0}
\pgfpathmoveto{\pgfqpoint{0.682cm}{0.671cm}}
\pgfpathlineto{\pgfqpoint{0.682cm}{0.042cm}}
\pgfusepath{stroke}
\end{pgfscope}
\end{pgfscope}
\end{pgfscope}
\end{pgfscope}
\end{tikzpicture}}}, X_{M, \varepsilon} )\\
     &  & + 3\lambda^2 X_{M, \varepsilon} \circ ( X_{M, \varepsilon}^{\!\resizebox{0.6em}{!}{
\begin{tikzpicture}
\pgfpathmoveto{\pgfqpoint{0cm}{-0.035cm}}
\pgfpathlineto{\pgfqpoint{1.376cm}{-0.035cm}}
\pgfpathlineto{\pgfqpoint{1.376cm}{1.552cm}}
\pgfpathlineto{\pgfqpoint{0cm}{1.552cm}}
\pgfpathclose
\pgfusepath{clip}
\begin{pgfscope}
\begin{pgfscope}
\pgfpathmoveto{\pgfqpoint{0cm}{-0.035cm}}
\pgfpathlineto{\pgfqpoint{1.376cm}{-0.035cm}}
\pgfpathlineto{\pgfqpoint{1.376cm}{1.552cm}}
\pgfpathlineto{\pgfqpoint{0cm}{1.552cm}}
\pgfpathclose
\pgfusepath{clip}
\begin{pgfscope}
\begin{pgfscope}
\pgfsetdash{}{0cm}
\pgfsetlinewidth{0.818mm}
\pgfsetroundcap
\pgfsetroundjoin
\pgfsetmiterlimit{7.0}
\definecolor{eps2pgf_color}{gray}{0}\pgfsetstrokecolor{eps2pgf_color}\pgfsetfillcolor{eps2pgf_color}
\pgfpathmoveto{\pgfqpoint{0.117cm}{1.421cm}}
\pgfpathlineto{\pgfqpoint{0.682cm}{0.671cm}}
\pgfpathlineto{\pgfqpoint{1.246cm}{1.421cm}}
\pgfusepath{stroke}
\end{pgfscope}
\definecolor{eps2pgf_color}{gray}{0}\pgfsetstrokecolor{eps2pgf_color}\pgfsetfillcolor{eps2pgf_color}
\pgfpathmoveto{\pgfqpoint{0.273cm}{1.395cm}}
\pgfpathcurveto{\pgfqpoint{0.273cm}{1.432cm}}{\pgfqpoint{0.259cm}{1.467cm}}{\pgfqpoint{0.233cm}{1.492cm}}
\pgfpathcurveto{\pgfqpoint{0.207cm}{1.518cm}}{\pgfqpoint{0.173cm}{1.532cm}}{\pgfqpoint{0.137cm}{1.532cm}}
\pgfpathcurveto{\pgfqpoint{0.1cm}{1.532cm}}{\pgfqpoint{0.066cm}{1.518cm}}{\pgfqpoint{0.04cm}{1.492cm}}
\pgfpathcurveto{\pgfqpoint{0.014cm}{1.467cm}}{\pgfqpoint{0cm}{1.432cm}}{\pgfqpoint{0cm}{1.395cm}}
\pgfpathcurveto{\pgfqpoint{0cm}{1.359cm}}{\pgfqpoint{0.014cm}{1.324cm}}{\pgfqpoint{0.04cm}{1.299cm}}
\pgfpathcurveto{\pgfqpoint{0.066cm}{1.273cm}}{\pgfqpoint{0.1cm}{1.258cm}}{\pgfqpoint{0.137cm}{1.258cm}}
\pgfpathcurveto{\pgfqpoint{0.173cm}{1.258cm}}{\pgfqpoint{0.207cm}{1.273cm}}{\pgfqpoint{0.233cm}{1.299cm}}
\pgfpathcurveto{\pgfqpoint{0.259cm}{1.324cm}}{\pgfqpoint{0.273cm}{1.359cm}}{\pgfqpoint{0.273cm}{1.395cm}}
\pgfusepath{fill}
\begin{pgfscope}
\pgfsetdash{}{0cm}
\pgfsetlinewidth{0.818mm}
\pgfsetmiterlimit{7.0}
\pgfpathmoveto{\pgfqpoint{0.682cm}{0.671cm}}
\pgfpathlineto{\pgfqpoint{0.679cm}{1.418cm}}
\pgfusepath{stroke}
\end{pgfscope}
\pgfpathmoveto{\pgfqpoint{0.815cm}{1.399cm}}
\pgfpathcurveto{\pgfqpoint{0.815cm}{1.435cm}}{\pgfqpoint{0.801cm}{1.47cm}}{\pgfqpoint{0.775cm}{1.496cm}}
\pgfpathcurveto{\pgfqpoint{0.75cm}{1.521cm}}{\pgfqpoint{0.715cm}{1.536cm}}{\pgfqpoint{0.679cm}{1.536cm}}
\pgfpathcurveto{\pgfqpoint{0.643cm}{1.536cm}}{\pgfqpoint{0.608cm}{1.521cm}}{\pgfqpoint{0.582cm}{1.496cm}}
\pgfpathcurveto{\pgfqpoint{0.557cm}{1.47cm}}{\pgfqpoint{0.542cm}{1.435cm}}{\pgfqpoint{0.542cm}{1.399cm}}
\pgfpathcurveto{\pgfqpoint{0.542cm}{1.363cm}}{\pgfqpoint{0.557cm}{1.328cm}}{\pgfqpoint{0.582cm}{1.302cm}}
\pgfpathcurveto{\pgfqpoint{0.608cm}{1.276cm}}{\pgfqpoint{0.643cm}{1.262cm}}{\pgfqpoint{0.679cm}{1.262cm}}
\pgfpathcurveto{\pgfqpoint{0.715cm}{1.262cm}}{\pgfqpoint{0.75cm}{1.276cm}}{\pgfqpoint{0.775cm}{1.302cm}}
\pgfpathcurveto{\pgfqpoint{0.801cm}{1.328cm}}{\pgfqpoint{0.815cm}{1.363cm}}{\pgfqpoint{0.815cm}{1.399cm}}
\pgfusepath{fill}
\pgfpathmoveto{\pgfqpoint{1.345cm}{1.371cm}}
\pgfpathcurveto{\pgfqpoint{1.345cm}{1.408cm}}{\pgfqpoint{1.331cm}{1.442cm}}{\pgfqpoint{1.305cm}{1.468cm}}
\pgfpathcurveto{\pgfqpoint{1.28cm}{1.494cm}}{\pgfqpoint{1.245cm}{1.508cm}}{\pgfqpoint{1.209cm}{1.508cm}}
\pgfpathcurveto{\pgfqpoint{1.172cm}{1.508cm}}{\pgfqpoint{1.138cm}{1.494cm}}{\pgfqpoint{1.112cm}{1.468cm}}
\pgfpathcurveto{\pgfqpoint{1.087cm}{1.442cm}}{\pgfqpoint{1.072cm}{1.408cm}}{\pgfqpoint{1.072cm}{1.371cm}}
\pgfpathcurveto{\pgfqpoint{1.072cm}{1.335cm}}{\pgfqpoint{1.087cm}{1.3cm}}{\pgfqpoint{1.112cm}{1.274cm}}
\pgfpathcurveto{\pgfqpoint{1.138cm}{1.249cm}}{\pgfqpoint{1.172cm}{1.234cm}}{\pgfqpoint{1.209cm}{1.234cm}}
\pgfpathcurveto{\pgfqpoint{1.245cm}{1.234cm}}{\pgfqpoint{1.28cm}{1.249cm}}{\pgfqpoint{1.305cm}{1.274cm}}
\pgfpathcurveto{\pgfqpoint{1.331cm}{1.3cm}}{\pgfqpoint{1.345cm}{1.335cm}}{\pgfqpoint{1.345cm}{1.371cm}}
\pgfusepath{fill}
\begin{pgfscope}
\pgfsetdash{}{0cm}
\pgfsetlinewidth{0.818mm}
\pgfsetroundcap
\pgfsetmiterlimit{4.0}
\pgfpathmoveto{\pgfqpoint{0.682cm}{0.671cm}}
\pgfpathlineto{\pgfqpoint{0.682cm}{0.042cm}}
\pgfusepath{stroke}
\end{pgfscope}
\end{pgfscope}
\end{pgfscope}
\end{pgfscope}
\end{tikzpicture}}}
     \circ X_{M, \varepsilon}^{\!\resizebox{0.6em}{!}{
\begin{tikzpicture}
\pgfpathmoveto{\pgfqpoint{0cm}{-0.035cm}}
\pgfpathlineto{\pgfqpoint{1.376cm}{-0.035cm}}
\pgfpathlineto{\pgfqpoint{1.376cm}{1.552cm}}
\pgfpathlineto{\pgfqpoint{0cm}{1.552cm}}
\pgfpathclose
\pgfusepath{clip}
\begin{pgfscope}
\begin{pgfscope}
\pgfpathmoveto{\pgfqpoint{0cm}{-0.035cm}}
\pgfpathlineto{\pgfqpoint{1.376cm}{-0.035cm}}
\pgfpathlineto{\pgfqpoint{1.376cm}{1.552cm}}
\pgfpathlineto{\pgfqpoint{0cm}{1.552cm}}
\pgfpathclose
\pgfusepath{clip}
\begin{pgfscope}
\begin{pgfscope}
\pgfsetdash{}{0cm}
\pgfsetlinewidth{0.818mm}
\pgfsetroundcap
\pgfsetroundjoin
\pgfsetmiterlimit{7.0}
\definecolor{eps2pgf_color}{gray}{0}\pgfsetstrokecolor{eps2pgf_color}\pgfsetfillcolor{eps2pgf_color}
\pgfpathmoveto{\pgfqpoint{0.117cm}{1.421cm}}
\pgfpathlineto{\pgfqpoint{0.682cm}{0.671cm}}
\pgfpathlineto{\pgfqpoint{1.246cm}{1.421cm}}
\pgfusepath{stroke}
\end{pgfscope}
\definecolor{eps2pgf_color}{gray}{0}\pgfsetstrokecolor{eps2pgf_color}\pgfsetfillcolor{eps2pgf_color}
\pgfpathmoveto{\pgfqpoint{0.273cm}{1.395cm}}
\pgfpathcurveto{\pgfqpoint{0.273cm}{1.432cm}}{\pgfqpoint{0.259cm}{1.467cm}}{\pgfqpoint{0.233cm}{1.492cm}}
\pgfpathcurveto{\pgfqpoint{0.207cm}{1.518cm}}{\pgfqpoint{0.173cm}{1.532cm}}{\pgfqpoint{0.137cm}{1.532cm}}
\pgfpathcurveto{\pgfqpoint{0.1cm}{1.532cm}}{\pgfqpoint{0.066cm}{1.518cm}}{\pgfqpoint{0.04cm}{1.492cm}}
\pgfpathcurveto{\pgfqpoint{0.014cm}{1.467cm}}{\pgfqpoint{0cm}{1.432cm}}{\pgfqpoint{0cm}{1.395cm}}
\pgfpathcurveto{\pgfqpoint{0cm}{1.359cm}}{\pgfqpoint{0.014cm}{1.324cm}}{\pgfqpoint{0.04cm}{1.299cm}}
\pgfpathcurveto{\pgfqpoint{0.066cm}{1.273cm}}{\pgfqpoint{0.1cm}{1.258cm}}{\pgfqpoint{0.137cm}{1.258cm}}
\pgfpathcurveto{\pgfqpoint{0.173cm}{1.258cm}}{\pgfqpoint{0.207cm}{1.273cm}}{\pgfqpoint{0.233cm}{1.299cm}}
\pgfpathcurveto{\pgfqpoint{0.259cm}{1.324cm}}{\pgfqpoint{0.273cm}{1.359cm}}{\pgfqpoint{0.273cm}{1.395cm}}
\pgfusepath{fill}
\begin{pgfscope}
\pgfsetdash{}{0cm}
\pgfsetlinewidth{0.818mm}
\pgfsetmiterlimit{7.0}
\pgfpathmoveto{\pgfqpoint{0.682cm}{0.671cm}}
\pgfpathlineto{\pgfqpoint{0.679cm}{1.418cm}}
\pgfusepath{stroke}
\end{pgfscope}
\pgfpathmoveto{\pgfqpoint{0.815cm}{1.399cm}}
\pgfpathcurveto{\pgfqpoint{0.815cm}{1.435cm}}{\pgfqpoint{0.801cm}{1.47cm}}{\pgfqpoint{0.775cm}{1.496cm}}
\pgfpathcurveto{\pgfqpoint{0.75cm}{1.521cm}}{\pgfqpoint{0.715cm}{1.536cm}}{\pgfqpoint{0.679cm}{1.536cm}}
\pgfpathcurveto{\pgfqpoint{0.643cm}{1.536cm}}{\pgfqpoint{0.608cm}{1.521cm}}{\pgfqpoint{0.582cm}{1.496cm}}
\pgfpathcurveto{\pgfqpoint{0.557cm}{1.47cm}}{\pgfqpoint{0.542cm}{1.435cm}}{\pgfqpoint{0.542cm}{1.399cm}}
\pgfpathcurveto{\pgfqpoint{0.542cm}{1.363cm}}{\pgfqpoint{0.557cm}{1.328cm}}{\pgfqpoint{0.582cm}{1.302cm}}
\pgfpathcurveto{\pgfqpoint{0.608cm}{1.276cm}}{\pgfqpoint{0.643cm}{1.262cm}}{\pgfqpoint{0.679cm}{1.262cm}}
\pgfpathcurveto{\pgfqpoint{0.715cm}{1.262cm}}{\pgfqpoint{0.75cm}{1.276cm}}{\pgfqpoint{0.775cm}{1.302cm}}
\pgfpathcurveto{\pgfqpoint{0.801cm}{1.328cm}}{\pgfqpoint{0.815cm}{1.363cm}}{\pgfqpoint{0.815cm}{1.399cm}}
\pgfusepath{fill}
\pgfpathmoveto{\pgfqpoint{1.345cm}{1.371cm}}
\pgfpathcurveto{\pgfqpoint{1.345cm}{1.408cm}}{\pgfqpoint{1.331cm}{1.442cm}}{\pgfqpoint{1.305cm}{1.468cm}}
\pgfpathcurveto{\pgfqpoint{1.28cm}{1.494cm}}{\pgfqpoint{1.245cm}{1.508cm}}{\pgfqpoint{1.209cm}{1.508cm}}
\pgfpathcurveto{\pgfqpoint{1.172cm}{1.508cm}}{\pgfqpoint{1.138cm}{1.494cm}}{\pgfqpoint{1.112cm}{1.468cm}}
\pgfpathcurveto{\pgfqpoint{1.087cm}{1.442cm}}{\pgfqpoint{1.072cm}{1.408cm}}{\pgfqpoint{1.072cm}{1.371cm}}
\pgfpathcurveto{\pgfqpoint{1.072cm}{1.335cm}}{\pgfqpoint{1.087cm}{1.3cm}}{\pgfqpoint{1.112cm}{1.274cm}}
\pgfpathcurveto{\pgfqpoint{1.138cm}{1.249cm}}{\pgfqpoint{1.172cm}{1.234cm}}{\pgfqpoint{1.209cm}{1.234cm}}
\pgfpathcurveto{\pgfqpoint{1.245cm}{1.234cm}}{\pgfqpoint{1.28cm}{1.249cm}}{\pgfqpoint{1.305cm}{1.274cm}}
\pgfpathcurveto{\pgfqpoint{1.331cm}{1.3cm}}{\pgfqpoint{1.345cm}{1.335cm}}{\pgfqpoint{1.345cm}{1.371cm}}
\pgfusepath{fill}
\begin{pgfscope}
\pgfsetdash{}{0cm}
\pgfsetlinewidth{0.818mm}
\pgfsetroundcap
\pgfsetmiterlimit{4.0}
\pgfpathmoveto{\pgfqpoint{0.682cm}{0.671cm}}
\pgfpathlineto{\pgfqpoint{0.682cm}{0.042cm}}
\pgfusepath{stroke}
\end{pgfscope}
\end{pgfscope}
\end{pgfscope}
\end{pgfscope}
\end{tikzpicture}}} ) - 6\lambda X_{M, \varepsilon} \circ
     ( X_{M, \varepsilon}^{\!\resizebox{0.6em}{!}{
\begin{tikzpicture}
\pgfpathmoveto{\pgfqpoint{0cm}{-0.035cm}}
\pgfpathlineto{\pgfqpoint{1.376cm}{-0.035cm}}
\pgfpathlineto{\pgfqpoint{1.376cm}{1.552cm}}
\pgfpathlineto{\pgfqpoint{0cm}{1.552cm}}
\pgfpathclose
\pgfusepath{clip}
\begin{pgfscope}
\begin{pgfscope}
\pgfpathmoveto{\pgfqpoint{0cm}{-0.035cm}}
\pgfpathlineto{\pgfqpoint{1.376cm}{-0.035cm}}
\pgfpathlineto{\pgfqpoint{1.376cm}{1.552cm}}
\pgfpathlineto{\pgfqpoint{0cm}{1.552cm}}
\pgfpathclose
\pgfusepath{clip}
\begin{pgfscope}
\begin{pgfscope}
\pgfsetdash{}{0cm}
\pgfsetlinewidth{0.818mm}
\pgfsetroundcap
\pgfsetroundjoin
\pgfsetmiterlimit{7.0}
\definecolor{eps2pgf_color}{gray}{0}\pgfsetstrokecolor{eps2pgf_color}\pgfsetfillcolor{eps2pgf_color}
\pgfpathmoveto{\pgfqpoint{0.117cm}{1.421cm}}
\pgfpathlineto{\pgfqpoint{0.682cm}{0.671cm}}
\pgfpathlineto{\pgfqpoint{1.246cm}{1.421cm}}
\pgfusepath{stroke}
\end{pgfscope}
\definecolor{eps2pgf_color}{gray}{0}\pgfsetstrokecolor{eps2pgf_color}\pgfsetfillcolor{eps2pgf_color}
\pgfpathmoveto{\pgfqpoint{0.273cm}{1.395cm}}
\pgfpathcurveto{\pgfqpoint{0.273cm}{1.432cm}}{\pgfqpoint{0.259cm}{1.467cm}}{\pgfqpoint{0.233cm}{1.492cm}}
\pgfpathcurveto{\pgfqpoint{0.207cm}{1.518cm}}{\pgfqpoint{0.173cm}{1.532cm}}{\pgfqpoint{0.137cm}{1.532cm}}
\pgfpathcurveto{\pgfqpoint{0.1cm}{1.532cm}}{\pgfqpoint{0.066cm}{1.518cm}}{\pgfqpoint{0.04cm}{1.492cm}}
\pgfpathcurveto{\pgfqpoint{0.014cm}{1.467cm}}{\pgfqpoint{0cm}{1.432cm}}{\pgfqpoint{0cm}{1.395cm}}
\pgfpathcurveto{\pgfqpoint{0cm}{1.359cm}}{\pgfqpoint{0.014cm}{1.324cm}}{\pgfqpoint{0.04cm}{1.299cm}}
\pgfpathcurveto{\pgfqpoint{0.066cm}{1.273cm}}{\pgfqpoint{0.1cm}{1.258cm}}{\pgfqpoint{0.137cm}{1.258cm}}
\pgfpathcurveto{\pgfqpoint{0.173cm}{1.258cm}}{\pgfqpoint{0.207cm}{1.273cm}}{\pgfqpoint{0.233cm}{1.299cm}}
\pgfpathcurveto{\pgfqpoint{0.259cm}{1.324cm}}{\pgfqpoint{0.273cm}{1.359cm}}{\pgfqpoint{0.273cm}{1.395cm}}
\pgfusepath{fill}
\begin{pgfscope}
\pgfsetdash{}{0cm}
\pgfsetlinewidth{0.818mm}
\pgfsetmiterlimit{7.0}
\pgfpathmoveto{\pgfqpoint{0.682cm}{0.671cm}}
\pgfpathlineto{\pgfqpoint{0.679cm}{1.418cm}}
\pgfusepath{stroke}
\end{pgfscope}
\pgfpathmoveto{\pgfqpoint{0.815cm}{1.399cm}}
\pgfpathcurveto{\pgfqpoint{0.815cm}{1.435cm}}{\pgfqpoint{0.801cm}{1.47cm}}{\pgfqpoint{0.775cm}{1.496cm}}
\pgfpathcurveto{\pgfqpoint{0.75cm}{1.521cm}}{\pgfqpoint{0.715cm}{1.536cm}}{\pgfqpoint{0.679cm}{1.536cm}}
\pgfpathcurveto{\pgfqpoint{0.643cm}{1.536cm}}{\pgfqpoint{0.608cm}{1.521cm}}{\pgfqpoint{0.582cm}{1.496cm}}
\pgfpathcurveto{\pgfqpoint{0.557cm}{1.47cm}}{\pgfqpoint{0.542cm}{1.435cm}}{\pgfqpoint{0.542cm}{1.399cm}}
\pgfpathcurveto{\pgfqpoint{0.542cm}{1.363cm}}{\pgfqpoint{0.557cm}{1.328cm}}{\pgfqpoint{0.582cm}{1.302cm}}
\pgfpathcurveto{\pgfqpoint{0.608cm}{1.276cm}}{\pgfqpoint{0.643cm}{1.262cm}}{\pgfqpoint{0.679cm}{1.262cm}}
\pgfpathcurveto{\pgfqpoint{0.715cm}{1.262cm}}{\pgfqpoint{0.75cm}{1.276cm}}{\pgfqpoint{0.775cm}{1.302cm}}
\pgfpathcurveto{\pgfqpoint{0.801cm}{1.328cm}}{\pgfqpoint{0.815cm}{1.363cm}}{\pgfqpoint{0.815cm}{1.399cm}}
\pgfusepath{fill}
\pgfpathmoveto{\pgfqpoint{1.345cm}{1.371cm}}
\pgfpathcurveto{\pgfqpoint{1.345cm}{1.408cm}}{\pgfqpoint{1.331cm}{1.442cm}}{\pgfqpoint{1.305cm}{1.468cm}}
\pgfpathcurveto{\pgfqpoint{1.28cm}{1.494cm}}{\pgfqpoint{1.245cm}{1.508cm}}{\pgfqpoint{1.209cm}{1.508cm}}
\pgfpathcurveto{\pgfqpoint{1.172cm}{1.508cm}}{\pgfqpoint{1.138cm}{1.494cm}}{\pgfqpoint{1.112cm}{1.468cm}}
\pgfpathcurveto{\pgfqpoint{1.087cm}{1.442cm}}{\pgfqpoint{1.072cm}{1.408cm}}{\pgfqpoint{1.072cm}{1.371cm}}
\pgfpathcurveto{\pgfqpoint{1.072cm}{1.335cm}}{\pgfqpoint{1.087cm}{1.3cm}}{\pgfqpoint{1.112cm}{1.274cm}}
\pgfpathcurveto{\pgfqpoint{1.138cm}{1.249cm}}{\pgfqpoint{1.172cm}{1.234cm}}{\pgfqpoint{1.209cm}{1.234cm}}
\pgfpathcurveto{\pgfqpoint{1.245cm}{1.234cm}}{\pgfqpoint{1.28cm}{1.249cm}}{\pgfqpoint{1.305cm}{1.274cm}}
\pgfpathcurveto{\pgfqpoint{1.331cm}{1.3cm}}{\pgfqpoint{1.345cm}{1.335cm}}{\pgfqpoint{1.345cm}{1.371cm}}
\pgfusepath{fill}
\begin{pgfscope}
\pgfsetdash{}{0cm}
\pgfsetlinewidth{0.818mm}
\pgfsetroundcap
\pgfsetmiterlimit{4.0}
\pgfpathmoveto{\pgfqpoint{0.682cm}{0.671cm}}
\pgfpathlineto{\pgfqpoint{0.682cm}{0.042cm}}
\pgfusepath{stroke}
\end{pgfscope}
\end{pgfscope}
\end{pgfscope}
\end{pgfscope}
\end{tikzpicture}}} \preccurlyeq \zeta_{M, \varepsilon}
     ) \\
     & &+ 3 X_{M, \varepsilon} \circ \zeta_{M, \varepsilon}^2,
   \end{array} \]
where we used the notation $f \preccurlyeq g = f \prec g + f \circ g$.

These decompositions and our estimates show that the products are all are
controlled in the space $L^1 (0, T, B^{- 1 - 3 \kappa, \varepsilon}_{1, 1}
(\rho^{4 + \sigma}))$. The term $\llbracket X^3_{M, \varepsilon} \rrbracket$
requires some care since it cannot be defined as a function of $t$. Indeed,
standard computations show that $\mathcal{E}^{\varepsilon} \llbracket X^3_{M,
\varepsilon} \rrbracket \rightarrow \llbracket X^3 \rrbracket$ in $W^{-
\kappa, \infty}_T \CC^{- 3 / 2 - \kappa, \varepsilon} (\rho^{\sigma})$,
namely, it requires just a mild regularization in time to be well defined and
it is the only one among the contributions to $\llbracket \varphi_{M,
\varepsilon}^3 \rrbracket$ which has negative time regularity. In particular,
we may write $\llbracket \varphi_{M, \varepsilon}^3 \rrbracket = \llbracket
X_{M, \varepsilon}^3 \rrbracket + H_{\varepsilon} (\varphi_{M, \varepsilon},
\mathbb{X}_{M, \varepsilon})$ where for $p \in [1, \infty)$
\[ \sup_{\varepsilon \in \mathcal{A}, M > 0} \mathbb{E} \| \llbracket X_{M,
   \varepsilon}^3 \rrbracket \|^{2 p}_{W^{- \kappa, \infty}_T \CC^{- 3 / 2 -
   \kappa, \varepsilon} (\rho^{\sigma})} + \sup_{\varepsilon \in \mathcal{A},
   M > 0} \mathbb{E} \| H_{\varepsilon} (\varphi_{M, \varepsilon},
   \mathbb{X}_{M, \varepsilon}) \|^{2 p}_{L^1_T B^{- 1 - 3 \kappa,
   \varepsilon}_{1, 1} (\rho^{4 + \sigma})} < \infty \]
is uniformly bounded in $M, \varepsilon$. The dependence of the function $H_{\varepsilon}$ on $\varepsilon$ comes from the corresponding dependence of the paraproducts as well as the resonant product on $\varepsilon$.

Now, let $h : \mathbb{R} \rightarrow \mathbb{R}$ be a smooth test function
with $\tmop{supp} h \subset [\tau, T]$ for some $0 < \tau < T < \infty$ and
such that $\int_{\mathbb{R}} h (t) \mathd t = 1$. Then by stationarity we can
rewrite the Littlewood--Paley blocks $\Delta^{\varepsilon}_j \mathcal{J}_{M,
\varepsilon} (F)$ as
\[ \Delta^{\varepsilon}_j \mathcal{J}_{M, \varepsilon} (F) = \int_{\mathbb{R}}
   h (t) \mathbb{E} [F (\mathcal{E}^{\varepsilon} \varphi_{M, \varepsilon}
   (t)) \Delta_j^{\varepsilon} \llbracket \varphi_{M, \varepsilon}^3 (t)
   \rrbracket_{M, \varepsilon}] \mathd t \]
\[ =\mathbb{E} \left[ \int_{\mathbb{R}} h (t) F (\mathcal{E}^{\varepsilon}
   \varphi_{M, \varepsilon} (t)) \Delta_j^{\varepsilon} \llbracket X_{M,
   \varepsilon}^3 \rrbracket (t) \mathd t \right] +\mathbb{E} \left[
   \int_{\mathbb{R}} h (t) F (\mathcal{E}^{\varepsilon} \varphi_{M,
   \varepsilon} (t)) \Delta_j^{\varepsilon} H_{\varepsilon} (\varphi_{M,
   \varepsilon}, \mathbb{X}_{M, \varepsilon}) (t) \mathd t \right] \]
\[ \backassign \Delta^{\varepsilon}_j \mathcal{J}^X_{M, \varepsilon} (F) +
   \Delta^{\varepsilon}_j \mathcal{J}^H_{M, \varepsilon} (F) . \]
As a consequence of Corollary~\ref{cor:t} and the discussion afterwards we
extract a subsequence converging in law and using the uniform bounds together with the $(\mathcal{E})$ property of our nonlinearities as defined on page 2073 in \cite{MP17},  we may
pass to the limit and conclude
\[ \lim_{\varepsilon \rightarrow 0, M \to \infty}
   \mathcal{E}^{\varepsilon} \mathcal{J}_{M, \varepsilon} (F)
   =\mathbb{E}_{\mu} \left[ \int_{\mathbb{R}} h (t) F (\varphi (t)) \llbracket
   \varphi^3 \rrbracket (t) \mathd t \right] \backassign \mathcal{J}_{\mu} (F)
   . \]
Here $\llbracket \varphi^3 \rrbracket$ is expressed (as $\llbracket
\varphi^3_{M, \varepsilon} \rrbracket$ before) as a measurable function of
$(\varphi, X)$ given by
\begin{equation}
  \begin{array}{lll}
    \llbracket \varphi^3 \rrbracket & \assign & \llbracket X^3 \rrbracket + 3
    \llbracket X^2 \rrbracket \Join (-\lambda X^{\!\resizebox{0.6em}{!}{
\begin{tikzpicture}
\pgfpathmoveto{\pgfqpoint{0cm}{-0.035cm}}
\pgfpathlineto{\pgfqpoint{1.376cm}{-0.035cm}}
\pgfpathlineto{\pgfqpoint{1.376cm}{1.552cm}}
\pgfpathlineto{\pgfqpoint{0cm}{1.552cm}}
\pgfpathclose
\pgfusepath{clip}
\begin{pgfscope}
\begin{pgfscope}
\pgfpathmoveto{\pgfqpoint{0cm}{-0.035cm}}
\pgfpathlineto{\pgfqpoint{1.376cm}{-0.035cm}}
\pgfpathlineto{\pgfqpoint{1.376cm}{1.552cm}}
\pgfpathlineto{\pgfqpoint{0cm}{1.552cm}}
\pgfpathclose
\pgfusepath{clip}
\begin{pgfscope}
\begin{pgfscope}
\pgfsetdash{}{0cm}
\pgfsetlinewidth{0.818mm}
\pgfsetroundcap
\pgfsetroundjoin
\pgfsetmiterlimit{7.0}
\definecolor{eps2pgf_color}{gray}{0}\pgfsetstrokecolor{eps2pgf_color}\pgfsetfillcolor{eps2pgf_color}
\pgfpathmoveto{\pgfqpoint{0.117cm}{1.421cm}}
\pgfpathlineto{\pgfqpoint{0.682cm}{0.671cm}}
\pgfpathlineto{\pgfqpoint{1.246cm}{1.421cm}}
\pgfusepath{stroke}
\end{pgfscope}
\definecolor{eps2pgf_color}{gray}{0}\pgfsetstrokecolor{eps2pgf_color}\pgfsetfillcolor{eps2pgf_color}
\pgfpathmoveto{\pgfqpoint{0.273cm}{1.395cm}}
\pgfpathcurveto{\pgfqpoint{0.273cm}{1.432cm}}{\pgfqpoint{0.259cm}{1.467cm}}{\pgfqpoint{0.233cm}{1.492cm}}
\pgfpathcurveto{\pgfqpoint{0.207cm}{1.518cm}}{\pgfqpoint{0.173cm}{1.532cm}}{\pgfqpoint{0.137cm}{1.532cm}}
\pgfpathcurveto{\pgfqpoint{0.1cm}{1.532cm}}{\pgfqpoint{0.066cm}{1.518cm}}{\pgfqpoint{0.04cm}{1.492cm}}
\pgfpathcurveto{\pgfqpoint{0.014cm}{1.467cm}}{\pgfqpoint{0cm}{1.432cm}}{\pgfqpoint{0cm}{1.395cm}}
\pgfpathcurveto{\pgfqpoint{0cm}{1.359cm}}{\pgfqpoint{0.014cm}{1.324cm}}{\pgfqpoint{0.04cm}{1.299cm}}
\pgfpathcurveto{\pgfqpoint{0.066cm}{1.273cm}}{\pgfqpoint{0.1cm}{1.258cm}}{\pgfqpoint{0.137cm}{1.258cm}}
\pgfpathcurveto{\pgfqpoint{0.173cm}{1.258cm}}{\pgfqpoint{0.207cm}{1.273cm}}{\pgfqpoint{0.233cm}{1.299cm}}
\pgfpathcurveto{\pgfqpoint{0.259cm}{1.324cm}}{\pgfqpoint{0.273cm}{1.359cm}}{\pgfqpoint{0.273cm}{1.395cm}}
\pgfusepath{fill}
\begin{pgfscope}
\pgfsetdash{}{0cm}
\pgfsetlinewidth{0.818mm}
\pgfsetmiterlimit{7.0}
\pgfpathmoveto{\pgfqpoint{0.682cm}{0.671cm}}
\pgfpathlineto{\pgfqpoint{0.679cm}{1.418cm}}
\pgfusepath{stroke}
\end{pgfscope}
\pgfpathmoveto{\pgfqpoint{0.815cm}{1.399cm}}
\pgfpathcurveto{\pgfqpoint{0.815cm}{1.435cm}}{\pgfqpoint{0.801cm}{1.47cm}}{\pgfqpoint{0.775cm}{1.496cm}}
\pgfpathcurveto{\pgfqpoint{0.75cm}{1.521cm}}{\pgfqpoint{0.715cm}{1.536cm}}{\pgfqpoint{0.679cm}{1.536cm}}
\pgfpathcurveto{\pgfqpoint{0.643cm}{1.536cm}}{\pgfqpoint{0.608cm}{1.521cm}}{\pgfqpoint{0.582cm}{1.496cm}}
\pgfpathcurveto{\pgfqpoint{0.557cm}{1.47cm}}{\pgfqpoint{0.542cm}{1.435cm}}{\pgfqpoint{0.542cm}{1.399cm}}
\pgfpathcurveto{\pgfqpoint{0.542cm}{1.363cm}}{\pgfqpoint{0.557cm}{1.328cm}}{\pgfqpoint{0.582cm}{1.302cm}}
\pgfpathcurveto{\pgfqpoint{0.608cm}{1.276cm}}{\pgfqpoint{0.643cm}{1.262cm}}{\pgfqpoint{0.679cm}{1.262cm}}
\pgfpathcurveto{\pgfqpoint{0.715cm}{1.262cm}}{\pgfqpoint{0.75cm}{1.276cm}}{\pgfqpoint{0.775cm}{1.302cm}}
\pgfpathcurveto{\pgfqpoint{0.801cm}{1.328cm}}{\pgfqpoint{0.815cm}{1.363cm}}{\pgfqpoint{0.815cm}{1.399cm}}
\pgfusepath{fill}
\pgfpathmoveto{\pgfqpoint{1.345cm}{1.371cm}}
\pgfpathcurveto{\pgfqpoint{1.345cm}{1.408cm}}{\pgfqpoint{1.331cm}{1.442cm}}{\pgfqpoint{1.305cm}{1.468cm}}
\pgfpathcurveto{\pgfqpoint{1.28cm}{1.494cm}}{\pgfqpoint{1.245cm}{1.508cm}}{\pgfqpoint{1.209cm}{1.508cm}}
\pgfpathcurveto{\pgfqpoint{1.172cm}{1.508cm}}{\pgfqpoint{1.138cm}{1.494cm}}{\pgfqpoint{1.112cm}{1.468cm}}
\pgfpathcurveto{\pgfqpoint{1.087cm}{1.442cm}}{\pgfqpoint{1.072cm}{1.408cm}}{\pgfqpoint{1.072cm}{1.371cm}}
\pgfpathcurveto{\pgfqpoint{1.072cm}{1.335cm}}{\pgfqpoint{1.087cm}{1.3cm}}{\pgfqpoint{1.112cm}{1.274cm}}
\pgfpathcurveto{\pgfqpoint{1.138cm}{1.249cm}}{\pgfqpoint{1.172cm}{1.234cm}}{\pgfqpoint{1.209cm}{1.234cm}}
\pgfpathcurveto{\pgfqpoint{1.245cm}{1.234cm}}{\pgfqpoint{1.28cm}{1.249cm}}{\pgfqpoint{1.305cm}{1.274cm}}
\pgfpathcurveto{\pgfqpoint{1.331cm}{1.3cm}}{\pgfqpoint{1.345cm}{1.335cm}}{\pgfqpoint{1.345cm}{1.371cm}}
\pgfusepath{fill}
\begin{pgfscope}
\pgfsetdash{}{0cm}
\pgfsetlinewidth{0.818mm}
\pgfsetroundcap
\pgfsetmiterlimit{4.0}
\pgfpathmoveto{\pgfqpoint{0.682cm}{0.671cm}}
\pgfpathlineto{\pgfqpoint{0.682cm}{0.042cm}}
\pgfusepath{stroke}
\end{pgfscope}
\end{pgfscope}
\end{pgfscope}
\end{pgfscope}
\end{tikzpicture}}} + \zeta) - \lambda Z -
  \lambda  \tilde{X}^{\!\resizebox{!}{.8em}{
\begin{tikzpicture}
\pgfpathmoveto{\pgfqpoint{0cm}{-0.035cm}}
\pgfpathlineto{\pgfqpoint{1.976cm}{-0.035cm}}
\pgfpathlineto{\pgfqpoint{1.976cm}{1.94cm}}
\pgfpathlineto{\pgfqpoint{0cm}{1.94cm}}
\pgfpathclose
\pgfusepath{clip}
\begin{pgfscope}
\begin{pgfscope}
\pgfpathmoveto{\pgfqpoint{0cm}{-0.035cm}}
\pgfpathlineto{\pgfqpoint{1.976cm}{-0.035cm}}
\pgfpathlineto{\pgfqpoint{1.976cm}{1.94cm}}
\pgfpathlineto{\pgfqpoint{0cm}{1.94cm}}
\pgfpathclose
\pgfusepath{clip}
\begin{pgfscope}
\begin{pgfscope}
\pgfsetdash{}{0cm}
\pgfsetlinewidth{0.818mm}
\pgfsetroundcap
\pgfsetroundjoin
\pgfsetmiterlimit{7.0}
\definecolor{eps2pgf_color}{gray}{0}\pgfsetstrokecolor{eps2pgf_color}\pgfsetfillcolor{eps2pgf_color}
\pgfpathmoveto{\pgfqpoint{0.117cm}{1.815cm}}
\pgfpathlineto{\pgfqpoint{0.682cm}{1.065cm}}
\pgfpathlineto{\pgfqpoint{1.246cm}{1.815cm}}
\pgfusepath{stroke}
\end{pgfscope}
\definecolor{eps2pgf_color}{gray}{0}\pgfsetstrokecolor{eps2pgf_color}\pgfsetfillcolor{eps2pgf_color}
\pgfpathmoveto{\pgfqpoint{0.273cm}{1.789cm}}
\pgfpathcurveto{\pgfqpoint{0.273cm}{1.825cm}}{\pgfqpoint{0.259cm}{1.86cm}}{\pgfqpoint{0.233cm}{1.886cm}}
\pgfpathcurveto{\pgfqpoint{0.207cm}{1.912cm}}{\pgfqpoint{0.173cm}{1.926cm}}{\pgfqpoint{0.137cm}{1.926cm}}
\pgfpathcurveto{\pgfqpoint{0.1cm}{1.926cm}}{\pgfqpoint{0.066cm}{1.912cm}}{\pgfqpoint{0.04cm}{1.886cm}}
\pgfpathcurveto{\pgfqpoint{0.014cm}{1.86cm}}{\pgfqpoint{0cm}{1.825cm}}{\pgfqpoint{0cm}{1.789cm}}
\pgfpathcurveto{\pgfqpoint{0cm}{1.753cm}}{\pgfqpoint{0.014cm}{1.718cm}}{\pgfqpoint{0.04cm}{1.692cm}}
\pgfpathcurveto{\pgfqpoint{0.066cm}{1.667cm}}{\pgfqpoint{0.1cm}{1.652cm}}{\pgfqpoint{0.137cm}{1.652cm}}
\pgfpathcurveto{\pgfqpoint{0.173cm}{1.652cm}}{\pgfqpoint{0.207cm}{1.667cm}}{\pgfqpoint{0.233cm}{1.692cm}}
\pgfpathcurveto{\pgfqpoint{0.259cm}{1.718cm}}{\pgfqpoint{0.273cm}{1.753cm}}{\pgfqpoint{0.273cm}{1.789cm}}
\pgfusepath{fill}
\pgfpathmoveto{\pgfqpoint{1.345cm}{1.765cm}}
\pgfpathcurveto{\pgfqpoint{1.345cm}{1.801cm}}{\pgfqpoint{1.331cm}{1.836cm}}{\pgfqpoint{1.305cm}{1.862cm}}
\pgfpathcurveto{\pgfqpoint{1.28cm}{1.887cm}}{\pgfqpoint{1.245cm}{1.902cm}}{\pgfqpoint{1.209cm}{1.902cm}}
\pgfpathcurveto{\pgfqpoint{1.172cm}{1.902cm}}{\pgfqpoint{1.138cm}{1.887cm}}{\pgfqpoint{1.112cm}{1.862cm}}
\pgfpathcurveto{\pgfqpoint{1.087cm}{1.836cm}}{\pgfqpoint{1.072cm}{1.801cm}}{\pgfqpoint{1.072cm}{1.765cm}}
\pgfpathcurveto{\pgfqpoint{1.072cm}{1.728cm}}{\pgfqpoint{1.087cm}{1.694cm}}{\pgfqpoint{1.112cm}{1.668cm}}
\pgfpathcurveto{\pgfqpoint{1.138cm}{1.642cm}}{\pgfqpoint{1.172cm}{1.628cm}}{\pgfqpoint{1.209cm}{1.628cm}}
\pgfpathcurveto{\pgfqpoint{1.245cm}{1.628cm}}{\pgfqpoint{1.28cm}{1.642cm}}{\pgfqpoint{1.305cm}{1.668cm}}
\pgfpathcurveto{\pgfqpoint{1.331cm}{1.694cm}}{\pgfqpoint{1.345cm}{1.728cm}}{\pgfqpoint{1.345cm}{1.765cm}}
\pgfusepath{fill}
\begin{pgfscope}
\pgfsetdash{}{0cm}
\pgfsetlinewidth{0.818mm}
\pgfsetroundcap
\pgfsetroundjoin
\pgfsetmiterlimit{7.0}
\pgfpathmoveto{\pgfqpoint{0.682cm}{1.065cm}}
\pgfpathlineto{\pgfqpoint{1.246cm}{0.315cm}}
\pgfpathlineto{\pgfqpoint{1.811cm}{1.065cm}}
\pgfusepath{stroke}
\end{pgfscope}
\pgfpathmoveto{\pgfqpoint{1.948cm}{1.065cm}}
\pgfpathcurveto{\pgfqpoint{1.948cm}{1.101cm}}{\pgfqpoint{1.933cm}{1.136cm}}{\pgfqpoint{1.907cm}{1.162cm}}
\pgfpathcurveto{\pgfqpoint{1.882cm}{1.187cm}}{\pgfqpoint{1.847cm}{1.202cm}}{\pgfqpoint{1.811cm}{1.202cm}}
\pgfpathcurveto{\pgfqpoint{1.775cm}{1.202cm}}{\pgfqpoint{1.74cm}{1.187cm}}{\pgfqpoint{1.714cm}{1.162cm}}
\pgfpathcurveto{\pgfqpoint{1.689cm}{1.136cm}}{\pgfqpoint{1.674cm}{1.101cm}}{\pgfqpoint{1.674cm}{1.065cm}}
\pgfpathcurveto{\pgfqpoint{1.674cm}{1.029cm}}{\pgfqpoint{1.689cm}{0.994cm}}{\pgfqpoint{1.714cm}{0.968cm}}
\pgfpathcurveto{\pgfqpoint{1.74cm}{0.942cm}}{\pgfqpoint{1.775cm}{0.928cm}}{\pgfqpoint{1.811cm}{0.928cm}}
\pgfpathcurveto{\pgfqpoint{1.847cm}{0.928cm}}{\pgfqpoint{1.882cm}{0.942cm}}{\pgfqpoint{1.907cm}{0.968cm}}
\pgfpathcurveto{\pgfqpoint{1.933cm}{0.994cm}}{\pgfqpoint{1.948cm}{1.029cm}}{\pgfqpoint{1.948cm}{1.065cm}}
\pgfusepath{fill}
\begin{pgfscope}
\pgfsetdash{}{0cm}
\pgfsetlinewidth{0.818mm}
\pgfsetmiterlimit{7.0}
\pgfpathmoveto{\pgfqpoint{1.246cm}{0.315cm}}
\pgfpathlineto{\pgfqpoint{1.244cm}{1.061cm}}
\pgfusepath{stroke}
\end{pgfscope}
\pgfpathmoveto{\pgfqpoint{1.38cm}{1.065cm}}
\pgfpathcurveto{\pgfqpoint{1.38cm}{1.101cm}}{\pgfqpoint{1.366cm}{1.136cm}}{\pgfqpoint{1.34cm}{1.162cm}}
\pgfpathcurveto{\pgfqpoint{1.315cm}{1.187cm}}{\pgfqpoint{1.28cm}{1.202cm}}{\pgfqpoint{1.244cm}{1.202cm}}
\pgfpathcurveto{\pgfqpoint{1.207cm}{1.202cm}}{\pgfqpoint{1.173cm}{1.187cm}}{\pgfqpoint{1.147cm}{1.162cm}}
\pgfpathcurveto{\pgfqpoint{1.121cm}{1.136cm}}{\pgfqpoint{1.107cm}{1.101cm}}{\pgfqpoint{1.107cm}{1.065cm}}
\pgfpathcurveto{\pgfqpoint{1.107cm}{1.029cm}}{\pgfqpoint{1.121cm}{0.994cm}}{\pgfqpoint{1.147cm}{0.968cm}}
\pgfpathcurveto{\pgfqpoint{1.173cm}{0.942cm}}{\pgfqpoint{1.207cm}{0.928cm}}{\pgfqpoint{1.244cm}{0.928cm}}
\pgfpathcurveto{\pgfqpoint{1.28cm}{0.928cm}}{\pgfqpoint{1.315cm}{0.942cm}}{\pgfqpoint{1.34cm}{0.968cm}}
\pgfpathcurveto{\pgfqpoint{1.366cm}{0.994cm}}{\pgfqpoint{1.38cm}{1.029cm}}{\pgfqpoint{1.38cm}{1.065cm}}
\pgfusepath{fill}
\begin{pgfscope}
\pgfsetdash{}{0cm}
\pgfsetlinewidth{0.818mm}
\pgfsetmiterlimit{4.0}
\pgfpathmoveto{\pgfqpoint{1.383cm}{0.178cm}}
\pgfpathcurveto{\pgfqpoint{1.383cm}{0.214cm}}{\pgfqpoint{1.369cm}{0.249cm}}{\pgfqpoint{1.343cm}{0.275cm}}
\pgfpathcurveto{\pgfqpoint{1.317cm}{0.3cm}}{\pgfqpoint{1.283cm}{0.315cm}}{\pgfqpoint{1.246cm}{0.315cm}}
\pgfpathcurveto{\pgfqpoint{1.21cm}{0.315cm}}{\pgfqpoint{1.175cm}{0.3cm}}{\pgfqpoint{1.15cm}{0.275cm}}
\pgfpathcurveto{\pgfqpoint{1.124cm}{0.249cm}}{\pgfqpoint{1.11cm}{0.214cm}}{\pgfqpoint{1.11cm}{0.178cm}}
\pgfpathcurveto{\pgfqpoint{1.11cm}{0.141cm}}{\pgfqpoint{1.124cm}{0.107cm}}{\pgfqpoint{1.15cm}{0.081cm}}
\pgfpathcurveto{\pgfqpoint{1.175cm}{0.055cm}}{\pgfqpoint{1.21cm}{0.041cm}}{\pgfqpoint{1.246cm}{0.041cm}}
\pgfpathcurveto{\pgfqpoint{1.283cm}{0.041cm}}{\pgfqpoint{1.317cm}{0.055cm}}{\pgfqpoint{1.343cm}{0.081cm}}
\pgfpathcurveto{\pgfqpoint{1.369cm}{0.107cm}}{\pgfqpoint{1.383cm}{0.141cm}}{\pgfqpoint{1.383cm}{0.178cm}}
\pgfusepath{stroke}
\end{pgfscope}
\end{pgfscope}
\end{pgfscope}
\end{pgfscope}
\end{tikzpicture}}} \phi + 3\lambda B(t) \phi\\
    &  & + \lambda C (\phi, - 3 X^{\!\resizebox{0.6em}{!}{
\begin{tikzpicture}
\pgfpathmoveto{\pgfqpoint{0cm}{0cm}}
\pgfpathlineto{\pgfqpoint{1.376cm}{0cm}}
\pgfpathlineto{\pgfqpoint{1.376cm}{1.588cm}}
\pgfpathlineto{\pgfqpoint{0cm}{1.588cm}}
\pgfpathclose
\pgfusepath{clip}
\begin{pgfscope}
\begin{pgfscope}
\pgfpathmoveto{\pgfqpoint{0cm}{0cm}}
\pgfpathlineto{\pgfqpoint{1.376cm}{0cm}}
\pgfpathlineto{\pgfqpoint{1.376cm}{1.588cm}}
\pgfpathlineto{\pgfqpoint{0cm}{1.588cm}}
\pgfpathclose
\pgfusepath{clip}
\begin{pgfscope}
\begin{pgfscope}
\definecolor{eps2pgf_color}{gray}{0.976471}\pgfsetstrokecolor{eps2pgf_color}\pgfsetfillcolor{eps2pgf_color}
\pgfpathmoveto{\pgfqpoint{0cm}{0cm}}
\pgfpathlineto{\pgfqpoint{1.376cm}{0cm}}
\pgfpathlineto{\pgfqpoint{1.376cm}{1.588cm}}
\pgfpathlineto{\pgfqpoint{0cm}{1.588cm}}
\pgfpathclose
\pgfusepath{fill}
\end{pgfscope}
\begin{pgfscope}
\pgfsetdash{}{0cm}
\pgfsetlinewidth{0.818mm}
\pgfsetroundcap
\pgfsetroundjoin
\pgfsetmiterlimit{7.0}
\definecolor{eps2pgf_color}{gray}{0}\pgfsetstrokecolor{eps2pgf_color}\pgfsetfillcolor{eps2pgf_color}
\pgfpathmoveto{\pgfqpoint{0.117cm}{1.476cm}}
\pgfpathlineto{\pgfqpoint{0.682cm}{0.726cm}}
\pgfpathlineto{\pgfqpoint{1.246cm}{1.476cm}}
\pgfusepath{stroke}
\end{pgfscope}
\definecolor{eps2pgf_color}{gray}{0}\pgfsetstrokecolor{eps2pgf_color}\pgfsetfillcolor{eps2pgf_color}
\pgfpathmoveto{\pgfqpoint{0.273cm}{1.451cm}}
\pgfpathcurveto{\pgfqpoint{0.273cm}{1.487cm}}{\pgfqpoint{0.259cm}{1.522cm}}{\pgfqpoint{0.233cm}{1.547cm}}
\pgfpathcurveto{\pgfqpoint{0.207cm}{1.573cm}}{\pgfqpoint{0.173cm}{1.588cm}}{\pgfqpoint{0.137cm}{1.588cm}}
\pgfpathcurveto{\pgfqpoint{0.1cm}{1.588cm}}{\pgfqpoint{0.066cm}{1.573cm}}{\pgfqpoint{0.04cm}{1.547cm}}
\pgfpathcurveto{\pgfqpoint{0.014cm}{1.522cm}}{\pgfqpoint{0cm}{1.487cm}}{\pgfqpoint{0cm}{1.451cm}}
\pgfpathcurveto{\pgfqpoint{0cm}{1.414cm}}{\pgfqpoint{0.014cm}{1.379cm}}{\pgfqpoint{0.04cm}{1.354cm}}
\pgfpathcurveto{\pgfqpoint{0.066cm}{1.328cm}}{\pgfqpoint{0.1cm}{1.314cm}}{\pgfqpoint{0.137cm}{1.314cm}}
\pgfpathcurveto{\pgfqpoint{0.173cm}{1.314cm}}{\pgfqpoint{0.207cm}{1.328cm}}{\pgfqpoint{0.233cm}{1.354cm}}
\pgfpathcurveto{\pgfqpoint{0.259cm}{1.379cm}}{\pgfqpoint{0.273cm}{1.414cm}}{\pgfqpoint{0.273cm}{1.451cm}}
\pgfusepath{fill}
\pgfpathmoveto{\pgfqpoint{1.345cm}{1.426cm}}
\pgfpathcurveto{\pgfqpoint{1.345cm}{1.463cm}}{\pgfqpoint{1.331cm}{1.497cm}}{\pgfqpoint{1.305cm}{1.523cm}}
\pgfpathcurveto{\pgfqpoint{1.28cm}{1.549cm}}{\pgfqpoint{1.245cm}{1.563cm}}{\pgfqpoint{1.209cm}{1.563cm}}
\pgfpathcurveto{\pgfqpoint{1.172cm}{1.563cm}}{\pgfqpoint{1.138cm}{1.549cm}}{\pgfqpoint{1.112cm}{1.523cm}}
\pgfpathcurveto{\pgfqpoint{1.087cm}{1.497cm}}{\pgfqpoint{1.072cm}{1.463cm}}{\pgfqpoint{1.072cm}{1.426cm}}
\pgfpathcurveto{\pgfqpoint{1.072cm}{1.39cm}}{\pgfqpoint{1.087cm}{1.355cm}}{\pgfqpoint{1.112cm}{1.329cm}}
\pgfpathcurveto{\pgfqpoint{1.138cm}{1.304cm}}{\pgfqpoint{1.172cm}{1.289cm}}{\pgfqpoint{1.209cm}{1.289cm}}
\pgfpathcurveto{\pgfqpoint{1.245cm}{1.289cm}}{\pgfqpoint{1.28cm}{1.304cm}}{\pgfqpoint{1.305cm}{1.329cm}}
\pgfpathcurveto{\pgfqpoint{1.331cm}{1.355cm}}{\pgfqpoint{1.345cm}{1.39cm}}{\pgfqpoint{1.345cm}{1.426cm}}
\pgfusepath{fill}
\begin{pgfscope}
\pgfsetdash{}{0cm}
\pgfsetlinewidth{0.818mm}
\pgfsetroundcap
\pgfsetmiterlimit{4.0}
\pgfpathmoveto{\pgfqpoint{0.682cm}{0.726cm}}
\pgfpathlineto{\pgfqpoint{0.682cm}{0.097cm}}
\pgfusepath{stroke}
\end{pgfscope}
\end{pgfscope}
\end{pgfscope}
\end{pgfscope}
\end{tikzpicture}}}, 3 \llbracket X^2 \rrbracket) + 3
    \llbracket X^2 \rrbracket \circ \chi + 3 X \Join (-\lambda X^{\!\resizebox{0.6em}{!}{
\begin{tikzpicture}
\pgfpathmoveto{\pgfqpoint{0cm}{-0.035cm}}
\pgfpathlineto{\pgfqpoint{1.376cm}{-0.035cm}}
\pgfpathlineto{\pgfqpoint{1.376cm}{1.552cm}}
\pgfpathlineto{\pgfqpoint{0cm}{1.552cm}}
\pgfpathclose
\pgfusepath{clip}
\begin{pgfscope}
\begin{pgfscope}
\pgfpathmoveto{\pgfqpoint{0cm}{-0.035cm}}
\pgfpathlineto{\pgfqpoint{1.376cm}{-0.035cm}}
\pgfpathlineto{\pgfqpoint{1.376cm}{1.552cm}}
\pgfpathlineto{\pgfqpoint{0cm}{1.552cm}}
\pgfpathclose
\pgfusepath{clip}
\begin{pgfscope}
\begin{pgfscope}
\pgfsetdash{}{0cm}
\pgfsetlinewidth{0.818mm}
\pgfsetroundcap
\pgfsetroundjoin
\pgfsetmiterlimit{7.0}
\definecolor{eps2pgf_color}{gray}{0}\pgfsetstrokecolor{eps2pgf_color}\pgfsetfillcolor{eps2pgf_color}
\pgfpathmoveto{\pgfqpoint{0.117cm}{1.421cm}}
\pgfpathlineto{\pgfqpoint{0.682cm}{0.671cm}}
\pgfpathlineto{\pgfqpoint{1.246cm}{1.421cm}}
\pgfusepath{stroke}
\end{pgfscope}
\definecolor{eps2pgf_color}{gray}{0}\pgfsetstrokecolor{eps2pgf_color}\pgfsetfillcolor{eps2pgf_color}
\pgfpathmoveto{\pgfqpoint{0.273cm}{1.395cm}}
\pgfpathcurveto{\pgfqpoint{0.273cm}{1.432cm}}{\pgfqpoint{0.259cm}{1.467cm}}{\pgfqpoint{0.233cm}{1.492cm}}
\pgfpathcurveto{\pgfqpoint{0.207cm}{1.518cm}}{\pgfqpoint{0.173cm}{1.532cm}}{\pgfqpoint{0.137cm}{1.532cm}}
\pgfpathcurveto{\pgfqpoint{0.1cm}{1.532cm}}{\pgfqpoint{0.066cm}{1.518cm}}{\pgfqpoint{0.04cm}{1.492cm}}
\pgfpathcurveto{\pgfqpoint{0.014cm}{1.467cm}}{\pgfqpoint{0cm}{1.432cm}}{\pgfqpoint{0cm}{1.395cm}}
\pgfpathcurveto{\pgfqpoint{0cm}{1.359cm}}{\pgfqpoint{0.014cm}{1.324cm}}{\pgfqpoint{0.04cm}{1.299cm}}
\pgfpathcurveto{\pgfqpoint{0.066cm}{1.273cm}}{\pgfqpoint{0.1cm}{1.258cm}}{\pgfqpoint{0.137cm}{1.258cm}}
\pgfpathcurveto{\pgfqpoint{0.173cm}{1.258cm}}{\pgfqpoint{0.207cm}{1.273cm}}{\pgfqpoint{0.233cm}{1.299cm}}
\pgfpathcurveto{\pgfqpoint{0.259cm}{1.324cm}}{\pgfqpoint{0.273cm}{1.359cm}}{\pgfqpoint{0.273cm}{1.395cm}}
\pgfusepath{fill}
\begin{pgfscope}
\pgfsetdash{}{0cm}
\pgfsetlinewidth{0.818mm}
\pgfsetmiterlimit{7.0}
\pgfpathmoveto{\pgfqpoint{0.682cm}{0.671cm}}
\pgfpathlineto{\pgfqpoint{0.679cm}{1.418cm}}
\pgfusepath{stroke}
\end{pgfscope}
\pgfpathmoveto{\pgfqpoint{0.815cm}{1.399cm}}
\pgfpathcurveto{\pgfqpoint{0.815cm}{1.435cm}}{\pgfqpoint{0.801cm}{1.47cm}}{\pgfqpoint{0.775cm}{1.496cm}}
\pgfpathcurveto{\pgfqpoint{0.75cm}{1.521cm}}{\pgfqpoint{0.715cm}{1.536cm}}{\pgfqpoint{0.679cm}{1.536cm}}
\pgfpathcurveto{\pgfqpoint{0.643cm}{1.536cm}}{\pgfqpoint{0.608cm}{1.521cm}}{\pgfqpoint{0.582cm}{1.496cm}}
\pgfpathcurveto{\pgfqpoint{0.557cm}{1.47cm}}{\pgfqpoint{0.542cm}{1.435cm}}{\pgfqpoint{0.542cm}{1.399cm}}
\pgfpathcurveto{\pgfqpoint{0.542cm}{1.363cm}}{\pgfqpoint{0.557cm}{1.328cm}}{\pgfqpoint{0.582cm}{1.302cm}}
\pgfpathcurveto{\pgfqpoint{0.608cm}{1.276cm}}{\pgfqpoint{0.643cm}{1.262cm}}{\pgfqpoint{0.679cm}{1.262cm}}
\pgfpathcurveto{\pgfqpoint{0.715cm}{1.262cm}}{\pgfqpoint{0.75cm}{1.276cm}}{\pgfqpoint{0.775cm}{1.302cm}}
\pgfpathcurveto{\pgfqpoint{0.801cm}{1.328cm}}{\pgfqpoint{0.815cm}{1.363cm}}{\pgfqpoint{0.815cm}{1.399cm}}
\pgfusepath{fill}
\pgfpathmoveto{\pgfqpoint{1.345cm}{1.371cm}}
\pgfpathcurveto{\pgfqpoint{1.345cm}{1.408cm}}{\pgfqpoint{1.331cm}{1.442cm}}{\pgfqpoint{1.305cm}{1.468cm}}
\pgfpathcurveto{\pgfqpoint{1.28cm}{1.494cm}}{\pgfqpoint{1.245cm}{1.508cm}}{\pgfqpoint{1.209cm}{1.508cm}}
\pgfpathcurveto{\pgfqpoint{1.172cm}{1.508cm}}{\pgfqpoint{1.138cm}{1.494cm}}{\pgfqpoint{1.112cm}{1.468cm}}
\pgfpathcurveto{\pgfqpoint{1.087cm}{1.442cm}}{\pgfqpoint{1.072cm}{1.408cm}}{\pgfqpoint{1.072cm}{1.371cm}}
\pgfpathcurveto{\pgfqpoint{1.072cm}{1.335cm}}{\pgfqpoint{1.087cm}{1.3cm}}{\pgfqpoint{1.112cm}{1.274cm}}
\pgfpathcurveto{\pgfqpoint{1.138cm}{1.249cm}}{\pgfqpoint{1.172cm}{1.234cm}}{\pgfqpoint{1.209cm}{1.234cm}}
\pgfpathcurveto{\pgfqpoint{1.245cm}{1.234cm}}{\pgfqpoint{1.28cm}{1.249cm}}{\pgfqpoint{1.305cm}{1.274cm}}
\pgfpathcurveto{\pgfqpoint{1.331cm}{1.3cm}}{\pgfqpoint{1.345cm}{1.335cm}}{\pgfqpoint{1.345cm}{1.371cm}}
\pgfusepath{fill}
\begin{pgfscope}
\pgfsetdash{}{0cm}
\pgfsetlinewidth{0.818mm}
\pgfsetroundcap
\pgfsetmiterlimit{4.0}
\pgfpathmoveto{\pgfqpoint{0.682cm}{0.671cm}}
\pgfpathlineto{\pgfqpoint{0.682cm}{0.042cm}}
\pgfusepath{stroke}
\end{pgfscope}
\end{pgfscope}
\end{pgfscope}
\end{pgfscope}
\end{tikzpicture}}} +
    \zeta)^2 + 6\lambda (\lambdaX^{\!\resizebox{0.6em}{!}{
\begin{tikzpicture}
\pgfpathmoveto{\pgfqpoint{0cm}{-0.035cm}}
\pgfpathlineto{\pgfqpoint{1.376cm}{-0.035cm}}
\pgfpathlineto{\pgfqpoint{1.376cm}{1.552cm}}
\pgfpathlineto{\pgfqpoint{0cm}{1.552cm}}
\pgfpathclose
\pgfusepath{clip}
\begin{pgfscope}
\begin{pgfscope}
\pgfpathmoveto{\pgfqpoint{0cm}{-0.035cm}}
\pgfpathlineto{\pgfqpoint{1.376cm}{-0.035cm}}
\pgfpathlineto{\pgfqpoint{1.376cm}{1.552cm}}
\pgfpathlineto{\pgfqpoint{0cm}{1.552cm}}
\pgfpathclose
\pgfusepath{clip}
\begin{pgfscope}
\begin{pgfscope}
\pgfsetdash{}{0cm}
\pgfsetlinewidth{0.818mm}
\pgfsetroundcap
\pgfsetroundjoin
\pgfsetmiterlimit{7.0}
\definecolor{eps2pgf_color}{gray}{0}\pgfsetstrokecolor{eps2pgf_color}\pgfsetfillcolor{eps2pgf_color}
\pgfpathmoveto{\pgfqpoint{0.117cm}{1.421cm}}
\pgfpathlineto{\pgfqpoint{0.682cm}{0.671cm}}
\pgfpathlineto{\pgfqpoint{1.246cm}{1.421cm}}
\pgfusepath{stroke}
\end{pgfscope}
\definecolor{eps2pgf_color}{gray}{0}\pgfsetstrokecolor{eps2pgf_color}\pgfsetfillcolor{eps2pgf_color}
\pgfpathmoveto{\pgfqpoint{0.273cm}{1.395cm}}
\pgfpathcurveto{\pgfqpoint{0.273cm}{1.432cm}}{\pgfqpoint{0.259cm}{1.467cm}}{\pgfqpoint{0.233cm}{1.492cm}}
\pgfpathcurveto{\pgfqpoint{0.207cm}{1.518cm}}{\pgfqpoint{0.173cm}{1.532cm}}{\pgfqpoint{0.137cm}{1.532cm}}
\pgfpathcurveto{\pgfqpoint{0.1cm}{1.532cm}}{\pgfqpoint{0.066cm}{1.518cm}}{\pgfqpoint{0.04cm}{1.492cm}}
\pgfpathcurveto{\pgfqpoint{0.014cm}{1.467cm}}{\pgfqpoint{0cm}{1.432cm}}{\pgfqpoint{0cm}{1.395cm}}
\pgfpathcurveto{\pgfqpoint{0cm}{1.359cm}}{\pgfqpoint{0.014cm}{1.324cm}}{\pgfqpoint{0.04cm}{1.299cm}}
\pgfpathcurveto{\pgfqpoint{0.066cm}{1.273cm}}{\pgfqpoint{0.1cm}{1.258cm}}{\pgfqpoint{0.137cm}{1.258cm}}
\pgfpathcurveto{\pgfqpoint{0.173cm}{1.258cm}}{\pgfqpoint{0.207cm}{1.273cm}}{\pgfqpoint{0.233cm}{1.299cm}}
\pgfpathcurveto{\pgfqpoint{0.259cm}{1.324cm}}{\pgfqpoint{0.273cm}{1.359cm}}{\pgfqpoint{0.273cm}{1.395cm}}
\pgfusepath{fill}
\begin{pgfscope}
\pgfsetdash{}{0cm}
\pgfsetlinewidth{0.818mm}
\pgfsetmiterlimit{7.0}
\pgfpathmoveto{\pgfqpoint{0.682cm}{0.671cm}}
\pgfpathlineto{\pgfqpoint{0.679cm}{1.418cm}}
\pgfusepath{stroke}
\end{pgfscope}
\pgfpathmoveto{\pgfqpoint{0.815cm}{1.399cm}}
\pgfpathcurveto{\pgfqpoint{0.815cm}{1.435cm}}{\pgfqpoint{0.801cm}{1.47cm}}{\pgfqpoint{0.775cm}{1.496cm}}
\pgfpathcurveto{\pgfqpoint{0.75cm}{1.521cm}}{\pgfqpoint{0.715cm}{1.536cm}}{\pgfqpoint{0.679cm}{1.536cm}}
\pgfpathcurveto{\pgfqpoint{0.643cm}{1.536cm}}{\pgfqpoint{0.608cm}{1.521cm}}{\pgfqpoint{0.582cm}{1.496cm}}
\pgfpathcurveto{\pgfqpoint{0.557cm}{1.47cm}}{\pgfqpoint{0.542cm}{1.435cm}}{\pgfqpoint{0.542cm}{1.399cm}}
\pgfpathcurveto{\pgfqpoint{0.542cm}{1.363cm}}{\pgfqpoint{0.557cm}{1.328cm}}{\pgfqpoint{0.582cm}{1.302cm}}
\pgfpathcurveto{\pgfqpoint{0.608cm}{1.276cm}}{\pgfqpoint{0.643cm}{1.262cm}}{\pgfqpoint{0.679cm}{1.262cm}}
\pgfpathcurveto{\pgfqpoint{0.715cm}{1.262cm}}{\pgfqpoint{0.75cm}{1.276cm}}{\pgfqpoint{0.775cm}{1.302cm}}
\pgfpathcurveto{\pgfqpoint{0.801cm}{1.328cm}}{\pgfqpoint{0.815cm}{1.363cm}}{\pgfqpoint{0.815cm}{1.399cm}}
\pgfusepath{fill}
\pgfpathmoveto{\pgfqpoint{1.345cm}{1.371cm}}
\pgfpathcurveto{\pgfqpoint{1.345cm}{1.408cm}}{\pgfqpoint{1.331cm}{1.442cm}}{\pgfqpoint{1.305cm}{1.468cm}}
\pgfpathcurveto{\pgfqpoint{1.28cm}{1.494cm}}{\pgfqpoint{1.245cm}{1.508cm}}{\pgfqpoint{1.209cm}{1.508cm}}
\pgfpathcurveto{\pgfqpoint{1.172cm}{1.508cm}}{\pgfqpoint{1.138cm}{1.494cm}}{\pgfqpoint{1.112cm}{1.468cm}}
\pgfpathcurveto{\pgfqpoint{1.087cm}{1.442cm}}{\pgfqpoint{1.072cm}{1.408cm}}{\pgfqpoint{1.072cm}{1.371cm}}
\pgfpathcurveto{\pgfqpoint{1.072cm}{1.335cm}}{\pgfqpoint{1.087cm}{1.3cm}}{\pgfqpoint{1.112cm}{1.274cm}}
\pgfpathcurveto{\pgfqpoint{1.138cm}{1.249cm}}{\pgfqpoint{1.172cm}{1.234cm}}{\pgfqpoint{1.209cm}{1.234cm}}
\pgfpathcurveto{\pgfqpoint{1.245cm}{1.234cm}}{\pgfqpoint{1.28cm}{1.249cm}}{\pgfqpoint{1.305cm}{1.274cm}}
\pgfpathcurveto{\pgfqpoint{1.331cm}{1.3cm}}{\pgfqpoint{1.345cm}{1.335cm}}{\pgfqpoint{1.345cm}{1.371cm}}
\pgfusepath{fill}
\begin{pgfscope}
\pgfsetdash{}{0cm}
\pgfsetlinewidth{0.818mm}
\pgfsetroundcap
\pgfsetmiterlimit{4.0}
\pgfpathmoveto{\pgfqpoint{0.682cm}{0.671cm}}
\pgfpathlineto{\pgfqpoint{0.682cm}{0.042cm}}
\pgfusepath{stroke}
\end{pgfscope}
\end{pgfscope}
\end{pgfscope}
\end{pgfscope}
\end{tikzpicture}}} - \zeta) X^{\!\resizebox{!}{.8em}{
\begin{tikzpicture}
\pgfpathmoveto{\pgfqpoint{0cm}{-0.035cm}}
\pgfpathlineto{\pgfqpoint{1.976cm}{-0.035cm}}
\pgfpathlineto{\pgfqpoint{1.976cm}{1.94cm}}
\pgfpathlineto{\pgfqpoint{0cm}{1.94cm}}
\pgfpathclose
\pgfusepath{clip}
\begin{pgfscope}
\begin{pgfscope}
\pgfpathmoveto{\pgfqpoint{0cm}{-0.035cm}}
\pgfpathlineto{\pgfqpoint{1.976cm}{-0.035cm}}
\pgfpathlineto{\pgfqpoint{1.976cm}{1.94cm}}
\pgfpathlineto{\pgfqpoint{0cm}{1.94cm}}
\pgfpathclose
\pgfusepath{clip}
\begin{pgfscope}
\begin{pgfscope}
\pgfsetdash{}{0cm}
\pgfsetlinewidth{0.818mm}
\pgfsetroundcap
\pgfsetroundjoin
\pgfsetmiterlimit{7.0}
\definecolor{eps2pgf_color}{gray}{0}\pgfsetstrokecolor{eps2pgf_color}\pgfsetfillcolor{eps2pgf_color}
\pgfpathmoveto{\pgfqpoint{0.117cm}{1.815cm}}
\pgfpathlineto{\pgfqpoint{0.682cm}{1.065cm}}
\pgfpathlineto{\pgfqpoint{1.246cm}{1.815cm}}
\pgfusepath{stroke}
\end{pgfscope}
\definecolor{eps2pgf_color}{gray}{0}\pgfsetstrokecolor{eps2pgf_color}\pgfsetfillcolor{eps2pgf_color}
\pgfpathmoveto{\pgfqpoint{0.273cm}{1.789cm}}
\pgfpathcurveto{\pgfqpoint{0.273cm}{1.825cm}}{\pgfqpoint{0.259cm}{1.86cm}}{\pgfqpoint{0.233cm}{1.886cm}}
\pgfpathcurveto{\pgfqpoint{0.207cm}{1.912cm}}{\pgfqpoint{0.173cm}{1.926cm}}{\pgfqpoint{0.137cm}{1.926cm}}
\pgfpathcurveto{\pgfqpoint{0.1cm}{1.926cm}}{\pgfqpoint{0.066cm}{1.912cm}}{\pgfqpoint{0.04cm}{1.886cm}}
\pgfpathcurveto{\pgfqpoint{0.014cm}{1.86cm}}{\pgfqpoint{0cm}{1.825cm}}{\pgfqpoint{0cm}{1.789cm}}
\pgfpathcurveto{\pgfqpoint{0cm}{1.753cm}}{\pgfqpoint{0.014cm}{1.718cm}}{\pgfqpoint{0.04cm}{1.692cm}}
\pgfpathcurveto{\pgfqpoint{0.066cm}{1.667cm}}{\pgfqpoint{0.1cm}{1.652cm}}{\pgfqpoint{0.137cm}{1.652cm}}
\pgfpathcurveto{\pgfqpoint{0.173cm}{1.652cm}}{\pgfqpoint{0.207cm}{1.667cm}}{\pgfqpoint{0.233cm}{1.692cm}}
\pgfpathcurveto{\pgfqpoint{0.259cm}{1.718cm}}{\pgfqpoint{0.273cm}{1.753cm}}{\pgfqpoint{0.273cm}{1.789cm}}
\pgfusepath{fill}
\begin{pgfscope}
\pgfsetdash{}{0cm}
\pgfsetlinewidth{0.818mm}
\pgfsetmiterlimit{7.0}
\pgfpathmoveto{\pgfqpoint{0.682cm}{1.065cm}}
\pgfpathlineto{\pgfqpoint{0.679cm}{1.812cm}}
\pgfusepath{stroke}
\end{pgfscope}
\pgfpathmoveto{\pgfqpoint{0.815cm}{1.793cm}}
\pgfpathcurveto{\pgfqpoint{0.815cm}{1.829cm}}{\pgfqpoint{0.801cm}{1.864cm}}{\pgfqpoint{0.775cm}{1.89cm}}
\pgfpathcurveto{\pgfqpoint{0.75cm}{1.915cm}}{\pgfqpoint{0.715cm}{1.93cm}}{\pgfqpoint{0.679cm}{1.93cm}}
\pgfpathcurveto{\pgfqpoint{0.643cm}{1.93cm}}{\pgfqpoint{0.608cm}{1.915cm}}{\pgfqpoint{0.582cm}{1.89cm}}
\pgfpathcurveto{\pgfqpoint{0.557cm}{1.864cm}}{\pgfqpoint{0.542cm}{1.829cm}}{\pgfqpoint{0.542cm}{1.793cm}}
\pgfpathcurveto{\pgfqpoint{0.542cm}{1.756cm}}{\pgfqpoint{0.557cm}{1.722cm}}{\pgfqpoint{0.582cm}{1.696cm}}
\pgfpathcurveto{\pgfqpoint{0.608cm}{1.67cm}}{\pgfqpoint{0.643cm}{1.656cm}}{\pgfqpoint{0.679cm}{1.656cm}}
\pgfpathcurveto{\pgfqpoint{0.715cm}{1.656cm}}{\pgfqpoint{0.75cm}{1.67cm}}{\pgfqpoint{0.775cm}{1.696cm}}
\pgfpathcurveto{\pgfqpoint{0.801cm}{1.722cm}}{\pgfqpoint{0.815cm}{1.756cm}}{\pgfqpoint{0.815cm}{1.793cm}}
\pgfusepath{fill}
\pgfpathmoveto{\pgfqpoint{1.345cm}{1.765cm}}
\pgfpathcurveto{\pgfqpoint{1.345cm}{1.801cm}}{\pgfqpoint{1.331cm}{1.836cm}}{\pgfqpoint{1.305cm}{1.862cm}}
\pgfpathcurveto{\pgfqpoint{1.28cm}{1.887cm}}{\pgfqpoint{1.245cm}{1.902cm}}{\pgfqpoint{1.209cm}{1.902cm}}
\pgfpathcurveto{\pgfqpoint{1.172cm}{1.902cm}}{\pgfqpoint{1.138cm}{1.887cm}}{\pgfqpoint{1.112cm}{1.862cm}}
\pgfpathcurveto{\pgfqpoint{1.087cm}{1.836cm}}{\pgfqpoint{1.072cm}{1.801cm}}{\pgfqpoint{1.072cm}{1.765cm}}
\pgfpathcurveto{\pgfqpoint{1.072cm}{1.728cm}}{\pgfqpoint{1.087cm}{1.694cm}}{\pgfqpoint{1.112cm}{1.668cm}}
\pgfpathcurveto{\pgfqpoint{1.138cm}{1.642cm}}{\pgfqpoint{1.172cm}{1.628cm}}{\pgfqpoint{1.209cm}{1.628cm}}
\pgfpathcurveto{\pgfqpoint{1.245cm}{1.628cm}}{\pgfqpoint{1.28cm}{1.642cm}}{\pgfqpoint{1.305cm}{1.668cm}}
\pgfpathcurveto{\pgfqpoint{1.331cm}{1.694cm}}{\pgfqpoint{1.345cm}{1.728cm}}{\pgfqpoint{1.345cm}{1.765cm}}
\pgfusepath{fill}
\begin{pgfscope}
\pgfsetdash{}{0cm}
\pgfsetlinewidth{0.818mm}
\pgfsetroundcap
\pgfsetroundjoin
\pgfsetmiterlimit{7.0}
\pgfpathmoveto{\pgfqpoint{0.682cm}{1.065cm}}
\pgfpathlineto{\pgfqpoint{1.246cm}{0.315cm}}
\pgfpathlineto{\pgfqpoint{1.811cm}{1.065cm}}
\pgfusepath{stroke}
\end{pgfscope}
\pgfpathmoveto{\pgfqpoint{1.948cm}{1.065cm}}
\pgfpathcurveto{\pgfqpoint{1.948cm}{1.101cm}}{\pgfqpoint{1.933cm}{1.136cm}}{\pgfqpoint{1.907cm}{1.162cm}}
\pgfpathcurveto{\pgfqpoint{1.882cm}{1.187cm}}{\pgfqpoint{1.847cm}{1.202cm}}{\pgfqpoint{1.811cm}{1.202cm}}
\pgfpathcurveto{\pgfqpoint{1.775cm}{1.202cm}}{\pgfqpoint{1.74cm}{1.187cm}}{\pgfqpoint{1.714cm}{1.162cm}}
\pgfpathcurveto{\pgfqpoint{1.689cm}{1.136cm}}{\pgfqpoint{1.674cm}{1.101cm}}{\pgfqpoint{1.674cm}{1.065cm}}
\pgfpathcurveto{\pgfqpoint{1.674cm}{1.029cm}}{\pgfqpoint{1.689cm}{0.994cm}}{\pgfqpoint{1.714cm}{0.968cm}}
\pgfpathcurveto{\pgfqpoint{1.74cm}{0.942cm}}{\pgfqpoint{1.775cm}{0.928cm}}{\pgfqpoint{1.811cm}{0.928cm}}
\pgfpathcurveto{\pgfqpoint{1.847cm}{0.928cm}}{\pgfqpoint{1.882cm}{0.942cm}}{\pgfqpoint{1.907cm}{0.968cm}}
\pgfpathcurveto{\pgfqpoint{1.933cm}{0.994cm}}{\pgfqpoint{1.948cm}{1.029cm}}{\pgfqpoint{1.948cm}{1.065cm}}
\pgfusepath{fill}
\begin{pgfscope}
\pgfsetdash{}{0cm}
\pgfsetlinewidth{0.818mm}
\pgfsetmiterlimit{4.0}
\pgfpathmoveto{\pgfqpoint{1.383cm}{0.178cm}}
\pgfpathcurveto{\pgfqpoint{1.383cm}{0.214cm}}{\pgfqpoint{1.369cm}{0.249cm}}{\pgfqpoint{1.343cm}{0.275cm}}
\pgfpathcurveto{\pgfqpoint{1.317cm}{0.3cm}}{\pgfqpoint{1.283cm}{0.315cm}}{\pgfqpoint{1.246cm}{0.315cm}}
\pgfpathcurveto{\pgfqpoint{1.21cm}{0.315cm}}{\pgfqpoint{1.175cm}{0.3cm}}{\pgfqpoint{1.15cm}{0.275cm}}
\pgfpathcurveto{\pgfqpoint{1.124cm}{0.249cm}}{\pgfqpoint{1.11cm}{0.214cm}}{\pgfqpoint{1.11cm}{0.178cm}}
\pgfpathcurveto{\pgfqpoint{1.11cm}{0.141cm}}{\pgfqpoint{1.124cm}{0.107cm}}{\pgfqpoint{1.15cm}{0.081cm}}
\pgfpathcurveto{\pgfqpoint{1.175cm}{0.055cm}}{\pgfqpoint{1.21cm}{0.041cm}}{\pgfqpoint{1.246cm}{0.041cm}}
\pgfpathcurveto{\pgfqpoint{1.283cm}{0.041cm}}{\pgfqpoint{1.317cm}{0.055cm}}{\pgfqpoint{1.343cm}{0.081cm}}
\pgfpathcurveto{\pgfqpoint{1.369cm}{0.107cm}}{\pgfqpoint{1.383cm}{0.141cm}}{\pgfqpoint{1.383cm}{0.178cm}}
\pgfusepath{stroke}
\end{pgfscope}
\end{pgfscope}
\end{pgfscope}
\end{pgfscope}
\end{tikzpicture}}}\\
    &  & + 6\lambda C (\lambda X^{\!\resizebox{0.6em}{!}{
\begin{tikzpicture}
\pgfpathmoveto{\pgfqpoint{0cm}{-0.035cm}}
\pgfpathlineto{\pgfqpoint{1.376cm}{-0.035cm}}
\pgfpathlineto{\pgfqpoint{1.376cm}{1.552cm}}
\pgfpathlineto{\pgfqpoint{0cm}{1.552cm}}
\pgfpathclose
\pgfusepath{clip}
\begin{pgfscope}
\begin{pgfscope}
\pgfpathmoveto{\pgfqpoint{0cm}{-0.035cm}}
\pgfpathlineto{\pgfqpoint{1.376cm}{-0.035cm}}
\pgfpathlineto{\pgfqpoint{1.376cm}{1.552cm}}
\pgfpathlineto{\pgfqpoint{0cm}{1.552cm}}
\pgfpathclose
\pgfusepath{clip}
\begin{pgfscope}
\begin{pgfscope}
\pgfsetdash{}{0cm}
\pgfsetlinewidth{0.818mm}
\pgfsetroundcap
\pgfsetroundjoin
\pgfsetmiterlimit{7.0}
\definecolor{eps2pgf_color}{gray}{0}\pgfsetstrokecolor{eps2pgf_color}\pgfsetfillcolor{eps2pgf_color}
\pgfpathmoveto{\pgfqpoint{0.117cm}{1.421cm}}
\pgfpathlineto{\pgfqpoint{0.682cm}{0.671cm}}
\pgfpathlineto{\pgfqpoint{1.246cm}{1.421cm}}
\pgfusepath{stroke}
\end{pgfscope}
\definecolor{eps2pgf_color}{gray}{0}\pgfsetstrokecolor{eps2pgf_color}\pgfsetfillcolor{eps2pgf_color}
\pgfpathmoveto{\pgfqpoint{0.273cm}{1.395cm}}
\pgfpathcurveto{\pgfqpoint{0.273cm}{1.432cm}}{\pgfqpoint{0.259cm}{1.467cm}}{\pgfqpoint{0.233cm}{1.492cm}}
\pgfpathcurveto{\pgfqpoint{0.207cm}{1.518cm}}{\pgfqpoint{0.173cm}{1.532cm}}{\pgfqpoint{0.137cm}{1.532cm}}
\pgfpathcurveto{\pgfqpoint{0.1cm}{1.532cm}}{\pgfqpoint{0.066cm}{1.518cm}}{\pgfqpoint{0.04cm}{1.492cm}}
\pgfpathcurveto{\pgfqpoint{0.014cm}{1.467cm}}{\pgfqpoint{0cm}{1.432cm}}{\pgfqpoint{0cm}{1.395cm}}
\pgfpathcurveto{\pgfqpoint{0cm}{1.359cm}}{\pgfqpoint{0.014cm}{1.324cm}}{\pgfqpoint{0.04cm}{1.299cm}}
\pgfpathcurveto{\pgfqpoint{0.066cm}{1.273cm}}{\pgfqpoint{0.1cm}{1.258cm}}{\pgfqpoint{0.137cm}{1.258cm}}
\pgfpathcurveto{\pgfqpoint{0.173cm}{1.258cm}}{\pgfqpoint{0.207cm}{1.273cm}}{\pgfqpoint{0.233cm}{1.299cm}}
\pgfpathcurveto{\pgfqpoint{0.259cm}{1.324cm}}{\pgfqpoint{0.273cm}{1.359cm}}{\pgfqpoint{0.273cm}{1.395cm}}
\pgfusepath{fill}
\begin{pgfscope}
\pgfsetdash{}{0cm}
\pgfsetlinewidth{0.818mm}
\pgfsetmiterlimit{7.0}
\pgfpathmoveto{\pgfqpoint{0.682cm}{0.671cm}}
\pgfpathlineto{\pgfqpoint{0.679cm}{1.418cm}}
\pgfusepath{stroke}
\end{pgfscope}
\pgfpathmoveto{\pgfqpoint{0.815cm}{1.399cm}}
\pgfpathcurveto{\pgfqpoint{0.815cm}{1.435cm}}{\pgfqpoint{0.801cm}{1.47cm}}{\pgfqpoint{0.775cm}{1.496cm}}
\pgfpathcurveto{\pgfqpoint{0.75cm}{1.521cm}}{\pgfqpoint{0.715cm}{1.536cm}}{\pgfqpoint{0.679cm}{1.536cm}}
\pgfpathcurveto{\pgfqpoint{0.643cm}{1.536cm}}{\pgfqpoint{0.608cm}{1.521cm}}{\pgfqpoint{0.582cm}{1.496cm}}
\pgfpathcurveto{\pgfqpoint{0.557cm}{1.47cm}}{\pgfqpoint{0.542cm}{1.435cm}}{\pgfqpoint{0.542cm}{1.399cm}}
\pgfpathcurveto{\pgfqpoint{0.542cm}{1.363cm}}{\pgfqpoint{0.557cm}{1.328cm}}{\pgfqpoint{0.582cm}{1.302cm}}
\pgfpathcurveto{\pgfqpoint{0.608cm}{1.276cm}}{\pgfqpoint{0.643cm}{1.262cm}}{\pgfqpoint{0.679cm}{1.262cm}}
\pgfpathcurveto{\pgfqpoint{0.715cm}{1.262cm}}{\pgfqpoint{0.75cm}{1.276cm}}{\pgfqpoint{0.775cm}{1.302cm}}
\pgfpathcurveto{\pgfqpoint{0.801cm}{1.328cm}}{\pgfqpoint{0.815cm}{1.363cm}}{\pgfqpoint{0.815cm}{1.399cm}}
\pgfusepath{fill}
\pgfpathmoveto{\pgfqpoint{1.345cm}{1.371cm}}
\pgfpathcurveto{\pgfqpoint{1.345cm}{1.408cm}}{\pgfqpoint{1.331cm}{1.442cm}}{\pgfqpoint{1.305cm}{1.468cm}}
\pgfpathcurveto{\pgfqpoint{1.28cm}{1.494cm}}{\pgfqpoint{1.245cm}{1.508cm}}{\pgfqpoint{1.209cm}{1.508cm}}
\pgfpathcurveto{\pgfqpoint{1.172cm}{1.508cm}}{\pgfqpoint{1.138cm}{1.494cm}}{\pgfqpoint{1.112cm}{1.468cm}}
\pgfpathcurveto{\pgfqpoint{1.087cm}{1.442cm}}{\pgfqpoint{1.072cm}{1.408cm}}{\pgfqpoint{1.072cm}{1.371cm}}
\pgfpathcurveto{\pgfqpoint{1.072cm}{1.335cm}}{\pgfqpoint{1.087cm}{1.3cm}}{\pgfqpoint{1.112cm}{1.274cm}}
\pgfpathcurveto{\pgfqpoint{1.138cm}{1.249cm}}{\pgfqpoint{1.172cm}{1.234cm}}{\pgfqpoint{1.209cm}{1.234cm}}
\pgfpathcurveto{\pgfqpoint{1.245cm}{1.234cm}}{\pgfqpoint{1.28cm}{1.249cm}}{\pgfqpoint{1.305cm}{1.274cm}}
\pgfpathcurveto{\pgfqpoint{1.331cm}{1.3cm}}{\pgfqpoint{1.345cm}{1.335cm}}{\pgfqpoint{1.345cm}{1.371cm}}
\pgfusepath{fill}
\begin{pgfscope}
\pgfsetdash{}{0cm}
\pgfsetlinewidth{0.818mm}
\pgfsetroundcap
\pgfsetmiterlimit{4.0}
\pgfpathmoveto{\pgfqpoint{0.682cm}{0.671cm}}
\pgfpathlineto{\pgfqpoint{0.682cm}{0.042cm}}
\pgfusepath{stroke}
\end{pgfscope}
\end{pgfscope}
\end{pgfscope}
\end{pgfscope}
\end{tikzpicture}}} - \zeta, X^{\!\resizebox{0.6em}{!}{
\begin{tikzpicture}
\pgfpathmoveto{\pgfqpoint{0cm}{-0.035cm}}
\pgfpathlineto{\pgfqpoint{1.376cm}{-0.035cm}}
\pgfpathlineto{\pgfqpoint{1.376cm}{1.552cm}}
\pgfpathlineto{\pgfqpoint{0cm}{1.552cm}}
\pgfpathclose
\pgfusepath{clip}
\begin{pgfscope}
\begin{pgfscope}
\pgfpathmoveto{\pgfqpoint{0cm}{-0.035cm}}
\pgfpathlineto{\pgfqpoint{1.376cm}{-0.035cm}}
\pgfpathlineto{\pgfqpoint{1.376cm}{1.552cm}}
\pgfpathlineto{\pgfqpoint{0cm}{1.552cm}}
\pgfpathclose
\pgfusepath{clip}
\begin{pgfscope}
\begin{pgfscope}
\pgfsetdash{}{0cm}
\pgfsetlinewidth{0.818mm}
\pgfsetroundcap
\pgfsetroundjoin
\pgfsetmiterlimit{7.0}
\definecolor{eps2pgf_color}{gray}{0}\pgfsetstrokecolor{eps2pgf_color}\pgfsetfillcolor{eps2pgf_color}
\pgfpathmoveto{\pgfqpoint{0.117cm}{1.421cm}}
\pgfpathlineto{\pgfqpoint{0.682cm}{0.671cm}}
\pgfpathlineto{\pgfqpoint{1.246cm}{1.421cm}}
\pgfusepath{stroke}
\end{pgfscope}
\definecolor{eps2pgf_color}{gray}{0}\pgfsetstrokecolor{eps2pgf_color}\pgfsetfillcolor{eps2pgf_color}
\pgfpathmoveto{\pgfqpoint{0.273cm}{1.395cm}}
\pgfpathcurveto{\pgfqpoint{0.273cm}{1.432cm}}{\pgfqpoint{0.259cm}{1.467cm}}{\pgfqpoint{0.233cm}{1.492cm}}
\pgfpathcurveto{\pgfqpoint{0.207cm}{1.518cm}}{\pgfqpoint{0.173cm}{1.532cm}}{\pgfqpoint{0.137cm}{1.532cm}}
\pgfpathcurveto{\pgfqpoint{0.1cm}{1.532cm}}{\pgfqpoint{0.066cm}{1.518cm}}{\pgfqpoint{0.04cm}{1.492cm}}
\pgfpathcurveto{\pgfqpoint{0.014cm}{1.467cm}}{\pgfqpoint{0cm}{1.432cm}}{\pgfqpoint{0cm}{1.395cm}}
\pgfpathcurveto{\pgfqpoint{0cm}{1.359cm}}{\pgfqpoint{0.014cm}{1.324cm}}{\pgfqpoint{0.04cm}{1.299cm}}
\pgfpathcurveto{\pgfqpoint{0.066cm}{1.273cm}}{\pgfqpoint{0.1cm}{1.258cm}}{\pgfqpoint{0.137cm}{1.258cm}}
\pgfpathcurveto{\pgfqpoint{0.173cm}{1.258cm}}{\pgfqpoint{0.207cm}{1.273cm}}{\pgfqpoint{0.233cm}{1.299cm}}
\pgfpathcurveto{\pgfqpoint{0.259cm}{1.324cm}}{\pgfqpoint{0.273cm}{1.359cm}}{\pgfqpoint{0.273cm}{1.395cm}}
\pgfusepath{fill}
\begin{pgfscope}
\pgfsetdash{}{0cm}
\pgfsetlinewidth{0.818mm}
\pgfsetmiterlimit{7.0}
\pgfpathmoveto{\pgfqpoint{0.682cm}{0.671cm}}
\pgfpathlineto{\pgfqpoint{0.679cm}{1.418cm}}
\pgfusepath{stroke}
\end{pgfscope}
\pgfpathmoveto{\pgfqpoint{0.815cm}{1.399cm}}
\pgfpathcurveto{\pgfqpoint{0.815cm}{1.435cm}}{\pgfqpoint{0.801cm}{1.47cm}}{\pgfqpoint{0.775cm}{1.496cm}}
\pgfpathcurveto{\pgfqpoint{0.75cm}{1.521cm}}{\pgfqpoint{0.715cm}{1.536cm}}{\pgfqpoint{0.679cm}{1.536cm}}
\pgfpathcurveto{\pgfqpoint{0.643cm}{1.536cm}}{\pgfqpoint{0.608cm}{1.521cm}}{\pgfqpoint{0.582cm}{1.496cm}}
\pgfpathcurveto{\pgfqpoint{0.557cm}{1.47cm}}{\pgfqpoint{0.542cm}{1.435cm}}{\pgfqpoint{0.542cm}{1.399cm}}
\pgfpathcurveto{\pgfqpoint{0.542cm}{1.363cm}}{\pgfqpoint{0.557cm}{1.328cm}}{\pgfqpoint{0.582cm}{1.302cm}}
\pgfpathcurveto{\pgfqpoint{0.608cm}{1.276cm}}{\pgfqpoint{0.643cm}{1.262cm}}{\pgfqpoint{0.679cm}{1.262cm}}
\pgfpathcurveto{\pgfqpoint{0.715cm}{1.262cm}}{\pgfqpoint{0.75cm}{1.276cm}}{\pgfqpoint{0.775cm}{1.302cm}}
\pgfpathcurveto{\pgfqpoint{0.801cm}{1.328cm}}{\pgfqpoint{0.815cm}{1.363cm}}{\pgfqpoint{0.815cm}{1.399cm}}
\pgfusepath{fill}
\pgfpathmoveto{\pgfqpoint{1.345cm}{1.371cm}}
\pgfpathcurveto{\pgfqpoint{1.345cm}{1.408cm}}{\pgfqpoint{1.331cm}{1.442cm}}{\pgfqpoint{1.305cm}{1.468cm}}
\pgfpathcurveto{\pgfqpoint{1.28cm}{1.494cm}}{\pgfqpoint{1.245cm}{1.508cm}}{\pgfqpoint{1.209cm}{1.508cm}}
\pgfpathcurveto{\pgfqpoint{1.172cm}{1.508cm}}{\pgfqpoint{1.138cm}{1.494cm}}{\pgfqpoint{1.112cm}{1.468cm}}
\pgfpathcurveto{\pgfqpoint{1.087cm}{1.442cm}}{\pgfqpoint{1.072cm}{1.408cm}}{\pgfqpoint{1.072cm}{1.371cm}}
\pgfpathcurveto{\pgfqpoint{1.072cm}{1.335cm}}{\pgfqpoint{1.087cm}{1.3cm}}{\pgfqpoint{1.112cm}{1.274cm}}
\pgfpathcurveto{\pgfqpoint{1.138cm}{1.249cm}}{\pgfqpoint{1.172cm}{1.234cm}}{\pgfqpoint{1.209cm}{1.234cm}}
\pgfpathcurveto{\pgfqpoint{1.245cm}{1.234cm}}{\pgfqpoint{1.28cm}{1.249cm}}{\pgfqpoint{1.305cm}{1.274cm}}
\pgfpathcurveto{\pgfqpoint{1.331cm}{1.3cm}}{\pgfqpoint{1.345cm}{1.335cm}}{\pgfqpoint{1.345cm}{1.371cm}}
\pgfusepath{fill}
\begin{pgfscope}
\pgfsetdash{}{0cm}
\pgfsetlinewidth{0.818mm}
\pgfsetroundcap
\pgfsetmiterlimit{4.0}
\pgfpathmoveto{\pgfqpoint{0.682cm}{0.671cm}}
\pgfpathlineto{\pgfqpoint{0.682cm}{0.042cm}}
\pgfusepath{stroke}
\end{pgfscope}
\end{pgfscope}
\end{pgfscope}
\end{pgfscope}
\end{tikzpicture}}}, X ) + 3\lambda^2 X
    \circ ( X^{\!\resizebox{0.6em}{!}{
\begin{tikzpicture}
\pgfpathmoveto{\pgfqpoint{0cm}{-0.035cm}}
\pgfpathlineto{\pgfqpoint{1.376cm}{-0.035cm}}
\pgfpathlineto{\pgfqpoint{1.376cm}{1.552cm}}
\pgfpathlineto{\pgfqpoint{0cm}{1.552cm}}
\pgfpathclose
\pgfusepath{clip}
\begin{pgfscope}
\begin{pgfscope}
\pgfpathmoveto{\pgfqpoint{0cm}{-0.035cm}}
\pgfpathlineto{\pgfqpoint{1.376cm}{-0.035cm}}
\pgfpathlineto{\pgfqpoint{1.376cm}{1.552cm}}
\pgfpathlineto{\pgfqpoint{0cm}{1.552cm}}
\pgfpathclose
\pgfusepath{clip}
\begin{pgfscope}
\begin{pgfscope}
\pgfsetdash{}{0cm}
\pgfsetlinewidth{0.818mm}
\pgfsetroundcap
\pgfsetroundjoin
\pgfsetmiterlimit{7.0}
\definecolor{eps2pgf_color}{gray}{0}\pgfsetstrokecolor{eps2pgf_color}\pgfsetfillcolor{eps2pgf_color}
\pgfpathmoveto{\pgfqpoint{0.117cm}{1.421cm}}
\pgfpathlineto{\pgfqpoint{0.682cm}{0.671cm}}
\pgfpathlineto{\pgfqpoint{1.246cm}{1.421cm}}
\pgfusepath{stroke}
\end{pgfscope}
\definecolor{eps2pgf_color}{gray}{0}\pgfsetstrokecolor{eps2pgf_color}\pgfsetfillcolor{eps2pgf_color}
\pgfpathmoveto{\pgfqpoint{0.273cm}{1.395cm}}
\pgfpathcurveto{\pgfqpoint{0.273cm}{1.432cm}}{\pgfqpoint{0.259cm}{1.467cm}}{\pgfqpoint{0.233cm}{1.492cm}}
\pgfpathcurveto{\pgfqpoint{0.207cm}{1.518cm}}{\pgfqpoint{0.173cm}{1.532cm}}{\pgfqpoint{0.137cm}{1.532cm}}
\pgfpathcurveto{\pgfqpoint{0.1cm}{1.532cm}}{\pgfqpoint{0.066cm}{1.518cm}}{\pgfqpoint{0.04cm}{1.492cm}}
\pgfpathcurveto{\pgfqpoint{0.014cm}{1.467cm}}{\pgfqpoint{0cm}{1.432cm}}{\pgfqpoint{0cm}{1.395cm}}
\pgfpathcurveto{\pgfqpoint{0cm}{1.359cm}}{\pgfqpoint{0.014cm}{1.324cm}}{\pgfqpoint{0.04cm}{1.299cm}}
\pgfpathcurveto{\pgfqpoint{0.066cm}{1.273cm}}{\pgfqpoint{0.1cm}{1.258cm}}{\pgfqpoint{0.137cm}{1.258cm}}
\pgfpathcurveto{\pgfqpoint{0.173cm}{1.258cm}}{\pgfqpoint{0.207cm}{1.273cm}}{\pgfqpoint{0.233cm}{1.299cm}}
\pgfpathcurveto{\pgfqpoint{0.259cm}{1.324cm}}{\pgfqpoint{0.273cm}{1.359cm}}{\pgfqpoint{0.273cm}{1.395cm}}
\pgfusepath{fill}
\begin{pgfscope}
\pgfsetdash{}{0cm}
\pgfsetlinewidth{0.818mm}
\pgfsetmiterlimit{7.0}
\pgfpathmoveto{\pgfqpoint{0.682cm}{0.671cm}}
\pgfpathlineto{\pgfqpoint{0.679cm}{1.418cm}}
\pgfusepath{stroke}
\end{pgfscope}
\pgfpathmoveto{\pgfqpoint{0.815cm}{1.399cm}}
\pgfpathcurveto{\pgfqpoint{0.815cm}{1.435cm}}{\pgfqpoint{0.801cm}{1.47cm}}{\pgfqpoint{0.775cm}{1.496cm}}
\pgfpathcurveto{\pgfqpoint{0.75cm}{1.521cm}}{\pgfqpoint{0.715cm}{1.536cm}}{\pgfqpoint{0.679cm}{1.536cm}}
\pgfpathcurveto{\pgfqpoint{0.643cm}{1.536cm}}{\pgfqpoint{0.608cm}{1.521cm}}{\pgfqpoint{0.582cm}{1.496cm}}
\pgfpathcurveto{\pgfqpoint{0.557cm}{1.47cm}}{\pgfqpoint{0.542cm}{1.435cm}}{\pgfqpoint{0.542cm}{1.399cm}}
\pgfpathcurveto{\pgfqpoint{0.542cm}{1.363cm}}{\pgfqpoint{0.557cm}{1.328cm}}{\pgfqpoint{0.582cm}{1.302cm}}
\pgfpathcurveto{\pgfqpoint{0.608cm}{1.276cm}}{\pgfqpoint{0.643cm}{1.262cm}}{\pgfqpoint{0.679cm}{1.262cm}}
\pgfpathcurveto{\pgfqpoint{0.715cm}{1.262cm}}{\pgfqpoint{0.75cm}{1.276cm}}{\pgfqpoint{0.775cm}{1.302cm}}
\pgfpathcurveto{\pgfqpoint{0.801cm}{1.328cm}}{\pgfqpoint{0.815cm}{1.363cm}}{\pgfqpoint{0.815cm}{1.399cm}}
\pgfusepath{fill}
\pgfpathmoveto{\pgfqpoint{1.345cm}{1.371cm}}
\pgfpathcurveto{\pgfqpoint{1.345cm}{1.408cm}}{\pgfqpoint{1.331cm}{1.442cm}}{\pgfqpoint{1.305cm}{1.468cm}}
\pgfpathcurveto{\pgfqpoint{1.28cm}{1.494cm}}{\pgfqpoint{1.245cm}{1.508cm}}{\pgfqpoint{1.209cm}{1.508cm}}
\pgfpathcurveto{\pgfqpoint{1.172cm}{1.508cm}}{\pgfqpoint{1.138cm}{1.494cm}}{\pgfqpoint{1.112cm}{1.468cm}}
\pgfpathcurveto{\pgfqpoint{1.087cm}{1.442cm}}{\pgfqpoint{1.072cm}{1.408cm}}{\pgfqpoint{1.072cm}{1.371cm}}
\pgfpathcurveto{\pgfqpoint{1.072cm}{1.335cm}}{\pgfqpoint{1.087cm}{1.3cm}}{\pgfqpoint{1.112cm}{1.274cm}}
\pgfpathcurveto{\pgfqpoint{1.138cm}{1.249cm}}{\pgfqpoint{1.172cm}{1.234cm}}{\pgfqpoint{1.209cm}{1.234cm}}
\pgfpathcurveto{\pgfqpoint{1.245cm}{1.234cm}}{\pgfqpoint{1.28cm}{1.249cm}}{\pgfqpoint{1.305cm}{1.274cm}}
\pgfpathcurveto{\pgfqpoint{1.331cm}{1.3cm}}{\pgfqpoint{1.345cm}{1.335cm}}{\pgfqpoint{1.345cm}{1.371cm}}
\pgfusepath{fill}
\begin{pgfscope}
\pgfsetdash{}{0cm}
\pgfsetlinewidth{0.818mm}
\pgfsetroundcap
\pgfsetmiterlimit{4.0}
\pgfpathmoveto{\pgfqpoint{0.682cm}{0.671cm}}
\pgfpathlineto{\pgfqpoint{0.682cm}{0.042cm}}
\pgfusepath{stroke}
\end{pgfscope}
\end{pgfscope}
\end{pgfscope}
\end{pgfscope}
\end{tikzpicture}}} \circ X^{\!\resizebox{0.6em}{!}{
\begin{tikzpicture}
\pgfpathmoveto{\pgfqpoint{0cm}{-0.035cm}}
\pgfpathlineto{\pgfqpoint{1.376cm}{-0.035cm}}
\pgfpathlineto{\pgfqpoint{1.376cm}{1.552cm}}
\pgfpathlineto{\pgfqpoint{0cm}{1.552cm}}
\pgfpathclose
\pgfusepath{clip}
\begin{pgfscope}
\begin{pgfscope}
\pgfpathmoveto{\pgfqpoint{0cm}{-0.035cm}}
\pgfpathlineto{\pgfqpoint{1.376cm}{-0.035cm}}
\pgfpathlineto{\pgfqpoint{1.376cm}{1.552cm}}
\pgfpathlineto{\pgfqpoint{0cm}{1.552cm}}
\pgfpathclose
\pgfusepath{clip}
\begin{pgfscope}
\begin{pgfscope}
\pgfsetdash{}{0cm}
\pgfsetlinewidth{0.818mm}
\pgfsetroundcap
\pgfsetroundjoin
\pgfsetmiterlimit{7.0}
\definecolor{eps2pgf_color}{gray}{0}\pgfsetstrokecolor{eps2pgf_color}\pgfsetfillcolor{eps2pgf_color}
\pgfpathmoveto{\pgfqpoint{0.117cm}{1.421cm}}
\pgfpathlineto{\pgfqpoint{0.682cm}{0.671cm}}
\pgfpathlineto{\pgfqpoint{1.246cm}{1.421cm}}
\pgfusepath{stroke}
\end{pgfscope}
\definecolor{eps2pgf_color}{gray}{0}\pgfsetstrokecolor{eps2pgf_color}\pgfsetfillcolor{eps2pgf_color}
\pgfpathmoveto{\pgfqpoint{0.273cm}{1.395cm}}
\pgfpathcurveto{\pgfqpoint{0.273cm}{1.432cm}}{\pgfqpoint{0.259cm}{1.467cm}}{\pgfqpoint{0.233cm}{1.492cm}}
\pgfpathcurveto{\pgfqpoint{0.207cm}{1.518cm}}{\pgfqpoint{0.173cm}{1.532cm}}{\pgfqpoint{0.137cm}{1.532cm}}
\pgfpathcurveto{\pgfqpoint{0.1cm}{1.532cm}}{\pgfqpoint{0.066cm}{1.518cm}}{\pgfqpoint{0.04cm}{1.492cm}}
\pgfpathcurveto{\pgfqpoint{0.014cm}{1.467cm}}{\pgfqpoint{0cm}{1.432cm}}{\pgfqpoint{0cm}{1.395cm}}
\pgfpathcurveto{\pgfqpoint{0cm}{1.359cm}}{\pgfqpoint{0.014cm}{1.324cm}}{\pgfqpoint{0.04cm}{1.299cm}}
\pgfpathcurveto{\pgfqpoint{0.066cm}{1.273cm}}{\pgfqpoint{0.1cm}{1.258cm}}{\pgfqpoint{0.137cm}{1.258cm}}
\pgfpathcurveto{\pgfqpoint{0.173cm}{1.258cm}}{\pgfqpoint{0.207cm}{1.273cm}}{\pgfqpoint{0.233cm}{1.299cm}}
\pgfpathcurveto{\pgfqpoint{0.259cm}{1.324cm}}{\pgfqpoint{0.273cm}{1.359cm}}{\pgfqpoint{0.273cm}{1.395cm}}
\pgfusepath{fill}
\begin{pgfscope}
\pgfsetdash{}{0cm}
\pgfsetlinewidth{0.818mm}
\pgfsetmiterlimit{7.0}
\pgfpathmoveto{\pgfqpoint{0.682cm}{0.671cm}}
\pgfpathlineto{\pgfqpoint{0.679cm}{1.418cm}}
\pgfusepath{stroke}
\end{pgfscope}
\pgfpathmoveto{\pgfqpoint{0.815cm}{1.399cm}}
\pgfpathcurveto{\pgfqpoint{0.815cm}{1.435cm}}{\pgfqpoint{0.801cm}{1.47cm}}{\pgfqpoint{0.775cm}{1.496cm}}
\pgfpathcurveto{\pgfqpoint{0.75cm}{1.521cm}}{\pgfqpoint{0.715cm}{1.536cm}}{\pgfqpoint{0.679cm}{1.536cm}}
\pgfpathcurveto{\pgfqpoint{0.643cm}{1.536cm}}{\pgfqpoint{0.608cm}{1.521cm}}{\pgfqpoint{0.582cm}{1.496cm}}
\pgfpathcurveto{\pgfqpoint{0.557cm}{1.47cm}}{\pgfqpoint{0.542cm}{1.435cm}}{\pgfqpoint{0.542cm}{1.399cm}}
\pgfpathcurveto{\pgfqpoint{0.542cm}{1.363cm}}{\pgfqpoint{0.557cm}{1.328cm}}{\pgfqpoint{0.582cm}{1.302cm}}
\pgfpathcurveto{\pgfqpoint{0.608cm}{1.276cm}}{\pgfqpoint{0.643cm}{1.262cm}}{\pgfqpoint{0.679cm}{1.262cm}}
\pgfpathcurveto{\pgfqpoint{0.715cm}{1.262cm}}{\pgfqpoint{0.75cm}{1.276cm}}{\pgfqpoint{0.775cm}{1.302cm}}
\pgfpathcurveto{\pgfqpoint{0.801cm}{1.328cm}}{\pgfqpoint{0.815cm}{1.363cm}}{\pgfqpoint{0.815cm}{1.399cm}}
\pgfusepath{fill}
\pgfpathmoveto{\pgfqpoint{1.345cm}{1.371cm}}
\pgfpathcurveto{\pgfqpoint{1.345cm}{1.408cm}}{\pgfqpoint{1.331cm}{1.442cm}}{\pgfqpoint{1.305cm}{1.468cm}}
\pgfpathcurveto{\pgfqpoint{1.28cm}{1.494cm}}{\pgfqpoint{1.245cm}{1.508cm}}{\pgfqpoint{1.209cm}{1.508cm}}
\pgfpathcurveto{\pgfqpoint{1.172cm}{1.508cm}}{\pgfqpoint{1.138cm}{1.494cm}}{\pgfqpoint{1.112cm}{1.468cm}}
\pgfpathcurveto{\pgfqpoint{1.087cm}{1.442cm}}{\pgfqpoint{1.072cm}{1.408cm}}{\pgfqpoint{1.072cm}{1.371cm}}
\pgfpathcurveto{\pgfqpoint{1.072cm}{1.335cm}}{\pgfqpoint{1.087cm}{1.3cm}}{\pgfqpoint{1.112cm}{1.274cm}}
\pgfpathcurveto{\pgfqpoint{1.138cm}{1.249cm}}{\pgfqpoint{1.172cm}{1.234cm}}{\pgfqpoint{1.209cm}{1.234cm}}
\pgfpathcurveto{\pgfqpoint{1.245cm}{1.234cm}}{\pgfqpoint{1.28cm}{1.249cm}}{\pgfqpoint{1.305cm}{1.274cm}}
\pgfpathcurveto{\pgfqpoint{1.331cm}{1.3cm}}{\pgfqpoint{1.345cm}{1.335cm}}{\pgfqpoint{1.345cm}{1.371cm}}
\pgfusepath{fill}
\begin{pgfscope}
\pgfsetdash{}{0cm}
\pgfsetlinewidth{0.818mm}
\pgfsetroundcap
\pgfsetmiterlimit{4.0}
\pgfpathmoveto{\pgfqpoint{0.682cm}{0.671cm}}
\pgfpathlineto{\pgfqpoint{0.682cm}{0.042cm}}
\pgfusepath{stroke}
\end{pgfscope}
\end{pgfscope}
\end{pgfscope}
\end{pgfscope}
\end{tikzpicture}}} ) - 6\lambda X \circ (
    X^{\!\resizebox{0.6em}{!}{
\begin{tikzpicture}
\pgfpathmoveto{\pgfqpoint{0cm}{-0.035cm}}
\pgfpathlineto{\pgfqpoint{1.376cm}{-0.035cm}}
\pgfpathlineto{\pgfqpoint{1.376cm}{1.552cm}}
\pgfpathlineto{\pgfqpoint{0cm}{1.552cm}}
\pgfpathclose
\pgfusepath{clip}
\begin{pgfscope}
\begin{pgfscope}
\pgfpathmoveto{\pgfqpoint{0cm}{-0.035cm}}
\pgfpathlineto{\pgfqpoint{1.376cm}{-0.035cm}}
\pgfpathlineto{\pgfqpoint{1.376cm}{1.552cm}}
\pgfpathlineto{\pgfqpoint{0cm}{1.552cm}}
\pgfpathclose
\pgfusepath{clip}
\begin{pgfscope}
\begin{pgfscope}
\pgfsetdash{}{0cm}
\pgfsetlinewidth{0.818mm}
\pgfsetroundcap
\pgfsetroundjoin
\pgfsetmiterlimit{7.0}
\definecolor{eps2pgf_color}{gray}{0}\pgfsetstrokecolor{eps2pgf_color}\pgfsetfillcolor{eps2pgf_color}
\pgfpathmoveto{\pgfqpoint{0.117cm}{1.421cm}}
\pgfpathlineto{\pgfqpoint{0.682cm}{0.671cm}}
\pgfpathlineto{\pgfqpoint{1.246cm}{1.421cm}}
\pgfusepath{stroke}
\end{pgfscope}
\definecolor{eps2pgf_color}{gray}{0}\pgfsetstrokecolor{eps2pgf_color}\pgfsetfillcolor{eps2pgf_color}
\pgfpathmoveto{\pgfqpoint{0.273cm}{1.395cm}}
\pgfpathcurveto{\pgfqpoint{0.273cm}{1.432cm}}{\pgfqpoint{0.259cm}{1.467cm}}{\pgfqpoint{0.233cm}{1.492cm}}
\pgfpathcurveto{\pgfqpoint{0.207cm}{1.518cm}}{\pgfqpoint{0.173cm}{1.532cm}}{\pgfqpoint{0.137cm}{1.532cm}}
\pgfpathcurveto{\pgfqpoint{0.1cm}{1.532cm}}{\pgfqpoint{0.066cm}{1.518cm}}{\pgfqpoint{0.04cm}{1.492cm}}
\pgfpathcurveto{\pgfqpoint{0.014cm}{1.467cm}}{\pgfqpoint{0cm}{1.432cm}}{\pgfqpoint{0cm}{1.395cm}}
\pgfpathcurveto{\pgfqpoint{0cm}{1.359cm}}{\pgfqpoint{0.014cm}{1.324cm}}{\pgfqpoint{0.04cm}{1.299cm}}
\pgfpathcurveto{\pgfqpoint{0.066cm}{1.273cm}}{\pgfqpoint{0.1cm}{1.258cm}}{\pgfqpoint{0.137cm}{1.258cm}}
\pgfpathcurveto{\pgfqpoint{0.173cm}{1.258cm}}{\pgfqpoint{0.207cm}{1.273cm}}{\pgfqpoint{0.233cm}{1.299cm}}
\pgfpathcurveto{\pgfqpoint{0.259cm}{1.324cm}}{\pgfqpoint{0.273cm}{1.359cm}}{\pgfqpoint{0.273cm}{1.395cm}}
\pgfusepath{fill}
\begin{pgfscope}
\pgfsetdash{}{0cm}
\pgfsetlinewidth{0.818mm}
\pgfsetmiterlimit{7.0}
\pgfpathmoveto{\pgfqpoint{0.682cm}{0.671cm}}
\pgfpathlineto{\pgfqpoint{0.679cm}{1.418cm}}
\pgfusepath{stroke}
\end{pgfscope}
\pgfpathmoveto{\pgfqpoint{0.815cm}{1.399cm}}
\pgfpathcurveto{\pgfqpoint{0.815cm}{1.435cm}}{\pgfqpoint{0.801cm}{1.47cm}}{\pgfqpoint{0.775cm}{1.496cm}}
\pgfpathcurveto{\pgfqpoint{0.75cm}{1.521cm}}{\pgfqpoint{0.715cm}{1.536cm}}{\pgfqpoint{0.679cm}{1.536cm}}
\pgfpathcurveto{\pgfqpoint{0.643cm}{1.536cm}}{\pgfqpoint{0.608cm}{1.521cm}}{\pgfqpoint{0.582cm}{1.496cm}}
\pgfpathcurveto{\pgfqpoint{0.557cm}{1.47cm}}{\pgfqpoint{0.542cm}{1.435cm}}{\pgfqpoint{0.542cm}{1.399cm}}
\pgfpathcurveto{\pgfqpoint{0.542cm}{1.363cm}}{\pgfqpoint{0.557cm}{1.328cm}}{\pgfqpoint{0.582cm}{1.302cm}}
\pgfpathcurveto{\pgfqpoint{0.608cm}{1.276cm}}{\pgfqpoint{0.643cm}{1.262cm}}{\pgfqpoint{0.679cm}{1.262cm}}
\pgfpathcurveto{\pgfqpoint{0.715cm}{1.262cm}}{\pgfqpoint{0.75cm}{1.276cm}}{\pgfqpoint{0.775cm}{1.302cm}}
\pgfpathcurveto{\pgfqpoint{0.801cm}{1.328cm}}{\pgfqpoint{0.815cm}{1.363cm}}{\pgfqpoint{0.815cm}{1.399cm}}
\pgfusepath{fill}
\pgfpathmoveto{\pgfqpoint{1.345cm}{1.371cm}}
\pgfpathcurveto{\pgfqpoint{1.345cm}{1.408cm}}{\pgfqpoint{1.331cm}{1.442cm}}{\pgfqpoint{1.305cm}{1.468cm}}
\pgfpathcurveto{\pgfqpoint{1.28cm}{1.494cm}}{\pgfqpoint{1.245cm}{1.508cm}}{\pgfqpoint{1.209cm}{1.508cm}}
\pgfpathcurveto{\pgfqpoint{1.172cm}{1.508cm}}{\pgfqpoint{1.138cm}{1.494cm}}{\pgfqpoint{1.112cm}{1.468cm}}
\pgfpathcurveto{\pgfqpoint{1.087cm}{1.442cm}}{\pgfqpoint{1.072cm}{1.408cm}}{\pgfqpoint{1.072cm}{1.371cm}}
\pgfpathcurveto{\pgfqpoint{1.072cm}{1.335cm}}{\pgfqpoint{1.087cm}{1.3cm}}{\pgfqpoint{1.112cm}{1.274cm}}
\pgfpathcurveto{\pgfqpoint{1.138cm}{1.249cm}}{\pgfqpoint{1.172cm}{1.234cm}}{\pgfqpoint{1.209cm}{1.234cm}}
\pgfpathcurveto{\pgfqpoint{1.245cm}{1.234cm}}{\pgfqpoint{1.28cm}{1.249cm}}{\pgfqpoint{1.305cm}{1.274cm}}
\pgfpathcurveto{\pgfqpoint{1.331cm}{1.3cm}}{\pgfqpoint{1.345cm}{1.335cm}}{\pgfqpoint{1.345cm}{1.371cm}}
\pgfusepath{fill}
\begin{pgfscope}
\pgfsetdash{}{0cm}
\pgfsetlinewidth{0.818mm}
\pgfsetroundcap
\pgfsetmiterlimit{4.0}
\pgfpathmoveto{\pgfqpoint{0.682cm}{0.671cm}}
\pgfpathlineto{\pgfqpoint{0.682cm}{0.042cm}}
\pgfusepath{stroke}
\end{pgfscope}
\end{pgfscope}
\end{pgfscope}
\end{pgfscope}
\end{tikzpicture}}} \preccurlyeq \zeta ) + 3 X \circ \zeta^2\\
    & &+ (-\lambda
    X^{\!\resizebox{0.6em}{!}{
\begin{tikzpicture}
\pgfpathmoveto{\pgfqpoint{0cm}{-0.035cm}}
\pgfpathlineto{\pgfqpoint{1.376cm}{-0.035cm}}
\pgfpathlineto{\pgfqpoint{1.376cm}{1.552cm}}
\pgfpathlineto{\pgfqpoint{0cm}{1.552cm}}
\pgfpathclose
\pgfusepath{clip}
\begin{pgfscope}
\begin{pgfscope}
\pgfpathmoveto{\pgfqpoint{0cm}{-0.035cm}}
\pgfpathlineto{\pgfqpoint{1.376cm}{-0.035cm}}
\pgfpathlineto{\pgfqpoint{1.376cm}{1.552cm}}
\pgfpathlineto{\pgfqpoint{0cm}{1.552cm}}
\pgfpathclose
\pgfusepath{clip}
\begin{pgfscope}
\begin{pgfscope}
\pgfsetdash{}{0cm}
\pgfsetlinewidth{0.818mm}
\pgfsetroundcap
\pgfsetroundjoin
\pgfsetmiterlimit{7.0}
\definecolor{eps2pgf_color}{gray}{0}\pgfsetstrokecolor{eps2pgf_color}\pgfsetfillcolor{eps2pgf_color}
\pgfpathmoveto{\pgfqpoint{0.117cm}{1.421cm}}
\pgfpathlineto{\pgfqpoint{0.682cm}{0.671cm}}
\pgfpathlineto{\pgfqpoint{1.246cm}{1.421cm}}
\pgfusepath{stroke}
\end{pgfscope}
\definecolor{eps2pgf_color}{gray}{0}\pgfsetstrokecolor{eps2pgf_color}\pgfsetfillcolor{eps2pgf_color}
\pgfpathmoveto{\pgfqpoint{0.273cm}{1.395cm}}
\pgfpathcurveto{\pgfqpoint{0.273cm}{1.432cm}}{\pgfqpoint{0.259cm}{1.467cm}}{\pgfqpoint{0.233cm}{1.492cm}}
\pgfpathcurveto{\pgfqpoint{0.207cm}{1.518cm}}{\pgfqpoint{0.173cm}{1.532cm}}{\pgfqpoint{0.137cm}{1.532cm}}
\pgfpathcurveto{\pgfqpoint{0.1cm}{1.532cm}}{\pgfqpoint{0.066cm}{1.518cm}}{\pgfqpoint{0.04cm}{1.492cm}}
\pgfpathcurveto{\pgfqpoint{0.014cm}{1.467cm}}{\pgfqpoint{0cm}{1.432cm}}{\pgfqpoint{0cm}{1.395cm}}
\pgfpathcurveto{\pgfqpoint{0cm}{1.359cm}}{\pgfqpoint{0.014cm}{1.324cm}}{\pgfqpoint{0.04cm}{1.299cm}}
\pgfpathcurveto{\pgfqpoint{0.066cm}{1.273cm}}{\pgfqpoint{0.1cm}{1.258cm}}{\pgfqpoint{0.137cm}{1.258cm}}
\pgfpathcurveto{\pgfqpoint{0.173cm}{1.258cm}}{\pgfqpoint{0.207cm}{1.273cm}}{\pgfqpoint{0.233cm}{1.299cm}}
\pgfpathcurveto{\pgfqpoint{0.259cm}{1.324cm}}{\pgfqpoint{0.273cm}{1.359cm}}{\pgfqpoint{0.273cm}{1.395cm}}
\pgfusepath{fill}
\begin{pgfscope}
\pgfsetdash{}{0cm}
\pgfsetlinewidth{0.818mm}
\pgfsetmiterlimit{7.0}
\pgfpathmoveto{\pgfqpoint{0.682cm}{0.671cm}}
\pgfpathlineto{\pgfqpoint{0.679cm}{1.418cm}}
\pgfusepath{stroke}
\end{pgfscope}
\pgfpathmoveto{\pgfqpoint{0.815cm}{1.399cm}}
\pgfpathcurveto{\pgfqpoint{0.815cm}{1.435cm}}{\pgfqpoint{0.801cm}{1.47cm}}{\pgfqpoint{0.775cm}{1.496cm}}
\pgfpathcurveto{\pgfqpoint{0.75cm}{1.521cm}}{\pgfqpoint{0.715cm}{1.536cm}}{\pgfqpoint{0.679cm}{1.536cm}}
\pgfpathcurveto{\pgfqpoint{0.643cm}{1.536cm}}{\pgfqpoint{0.608cm}{1.521cm}}{\pgfqpoint{0.582cm}{1.496cm}}
\pgfpathcurveto{\pgfqpoint{0.557cm}{1.47cm}}{\pgfqpoint{0.542cm}{1.435cm}}{\pgfqpoint{0.542cm}{1.399cm}}
\pgfpathcurveto{\pgfqpoint{0.542cm}{1.363cm}}{\pgfqpoint{0.557cm}{1.328cm}}{\pgfqpoint{0.582cm}{1.302cm}}
\pgfpathcurveto{\pgfqpoint{0.608cm}{1.276cm}}{\pgfqpoint{0.643cm}{1.262cm}}{\pgfqpoint{0.679cm}{1.262cm}}
\pgfpathcurveto{\pgfqpoint{0.715cm}{1.262cm}}{\pgfqpoint{0.75cm}{1.276cm}}{\pgfqpoint{0.775cm}{1.302cm}}
\pgfpathcurveto{\pgfqpoint{0.801cm}{1.328cm}}{\pgfqpoint{0.815cm}{1.363cm}}{\pgfqpoint{0.815cm}{1.399cm}}
\pgfusepath{fill}
\pgfpathmoveto{\pgfqpoint{1.345cm}{1.371cm}}
\pgfpathcurveto{\pgfqpoint{1.345cm}{1.408cm}}{\pgfqpoint{1.331cm}{1.442cm}}{\pgfqpoint{1.305cm}{1.468cm}}
\pgfpathcurveto{\pgfqpoint{1.28cm}{1.494cm}}{\pgfqpoint{1.245cm}{1.508cm}}{\pgfqpoint{1.209cm}{1.508cm}}
\pgfpathcurveto{\pgfqpoint{1.172cm}{1.508cm}}{\pgfqpoint{1.138cm}{1.494cm}}{\pgfqpoint{1.112cm}{1.468cm}}
\pgfpathcurveto{\pgfqpoint{1.087cm}{1.442cm}}{\pgfqpoint{1.072cm}{1.408cm}}{\pgfqpoint{1.072cm}{1.371cm}}
\pgfpathcurveto{\pgfqpoint{1.072cm}{1.335cm}}{\pgfqpoint{1.087cm}{1.3cm}}{\pgfqpoint{1.112cm}{1.274cm}}
\pgfpathcurveto{\pgfqpoint{1.138cm}{1.249cm}}{\pgfqpoint{1.172cm}{1.234cm}}{\pgfqpoint{1.209cm}{1.234cm}}
\pgfpathcurveto{\pgfqpoint{1.245cm}{1.234cm}}{\pgfqpoint{1.28cm}{1.249cm}}{\pgfqpoint{1.305cm}{1.274cm}}
\pgfpathcurveto{\pgfqpoint{1.331cm}{1.3cm}}{\pgfqpoint{1.345cm}{1.335cm}}{\pgfqpoint{1.345cm}{1.371cm}}
\pgfusepath{fill}
\begin{pgfscope}
\pgfsetdash{}{0cm}
\pgfsetlinewidth{0.818mm}
\pgfsetroundcap
\pgfsetmiterlimit{4.0}
\pgfpathmoveto{\pgfqpoint{0.682cm}{0.671cm}}
\pgfpathlineto{\pgfqpoint{0.682cm}{0.042cm}}
\pgfusepath{stroke}
\end{pgfscope}
\end{pgfscope}
\end{pgfscope}
\end{pgfscope}
\end{tikzpicture}}} + \zeta)^3,
  \end{array} \label{eq:phi3}
\end{equation}
where we used the notation $f \Join g = f \prec g + f \succ g$ and $\zeta,
\phi, Y$ are defined as starting from $(\varphi, \mathbb{X}) = \Psi (\varphi,
X)$ as
\[ \varphi = X -\lambda X^{\!\resizebox{0.6em}{!}{
\begin{tikzpicture}
\pgfpathmoveto{\pgfqpoint{0cm}{-0.035cm}}
\pgfpathlineto{\pgfqpoint{1.376cm}{-0.035cm}}
\pgfpathlineto{\pgfqpoint{1.376cm}{1.552cm}}
\pgfpathlineto{\pgfqpoint{0cm}{1.552cm}}
\pgfpathclose
\pgfusepath{clip}
\begin{pgfscope}
\begin{pgfscope}
\pgfpathmoveto{\pgfqpoint{0cm}{-0.035cm}}
\pgfpathlineto{\pgfqpoint{1.376cm}{-0.035cm}}
\pgfpathlineto{\pgfqpoint{1.376cm}{1.552cm}}
\pgfpathlineto{\pgfqpoint{0cm}{1.552cm}}
\pgfpathclose
\pgfusepath{clip}
\begin{pgfscope}
\begin{pgfscope}
\pgfsetdash{}{0cm}
\pgfsetlinewidth{0.818mm}
\pgfsetroundcap
\pgfsetroundjoin
\pgfsetmiterlimit{7.0}
\definecolor{eps2pgf_color}{gray}{0}\pgfsetstrokecolor{eps2pgf_color}\pgfsetfillcolor{eps2pgf_color}
\pgfpathmoveto{\pgfqpoint{0.117cm}{1.421cm}}
\pgfpathlineto{\pgfqpoint{0.682cm}{0.671cm}}
\pgfpathlineto{\pgfqpoint{1.246cm}{1.421cm}}
\pgfusepath{stroke}
\end{pgfscope}
\definecolor{eps2pgf_color}{gray}{0}\pgfsetstrokecolor{eps2pgf_color}\pgfsetfillcolor{eps2pgf_color}
\pgfpathmoveto{\pgfqpoint{0.273cm}{1.395cm}}
\pgfpathcurveto{\pgfqpoint{0.273cm}{1.432cm}}{\pgfqpoint{0.259cm}{1.467cm}}{\pgfqpoint{0.233cm}{1.492cm}}
\pgfpathcurveto{\pgfqpoint{0.207cm}{1.518cm}}{\pgfqpoint{0.173cm}{1.532cm}}{\pgfqpoint{0.137cm}{1.532cm}}
\pgfpathcurveto{\pgfqpoint{0.1cm}{1.532cm}}{\pgfqpoint{0.066cm}{1.518cm}}{\pgfqpoint{0.04cm}{1.492cm}}
\pgfpathcurveto{\pgfqpoint{0.014cm}{1.467cm}}{\pgfqpoint{0cm}{1.432cm}}{\pgfqpoint{0cm}{1.395cm}}
\pgfpathcurveto{\pgfqpoint{0cm}{1.359cm}}{\pgfqpoint{0.014cm}{1.324cm}}{\pgfqpoint{0.04cm}{1.299cm}}
\pgfpathcurveto{\pgfqpoint{0.066cm}{1.273cm}}{\pgfqpoint{0.1cm}{1.258cm}}{\pgfqpoint{0.137cm}{1.258cm}}
\pgfpathcurveto{\pgfqpoint{0.173cm}{1.258cm}}{\pgfqpoint{0.207cm}{1.273cm}}{\pgfqpoint{0.233cm}{1.299cm}}
\pgfpathcurveto{\pgfqpoint{0.259cm}{1.324cm}}{\pgfqpoint{0.273cm}{1.359cm}}{\pgfqpoint{0.273cm}{1.395cm}}
\pgfusepath{fill}
\begin{pgfscope}
\pgfsetdash{}{0cm}
\pgfsetlinewidth{0.818mm}
\pgfsetmiterlimit{7.0}
\pgfpathmoveto{\pgfqpoint{0.682cm}{0.671cm}}
\pgfpathlineto{\pgfqpoint{0.679cm}{1.418cm}}
\pgfusepath{stroke}
\end{pgfscope}
\pgfpathmoveto{\pgfqpoint{0.815cm}{1.399cm}}
\pgfpathcurveto{\pgfqpoint{0.815cm}{1.435cm}}{\pgfqpoint{0.801cm}{1.47cm}}{\pgfqpoint{0.775cm}{1.496cm}}
\pgfpathcurveto{\pgfqpoint{0.75cm}{1.521cm}}{\pgfqpoint{0.715cm}{1.536cm}}{\pgfqpoint{0.679cm}{1.536cm}}
\pgfpathcurveto{\pgfqpoint{0.643cm}{1.536cm}}{\pgfqpoint{0.608cm}{1.521cm}}{\pgfqpoint{0.582cm}{1.496cm}}
\pgfpathcurveto{\pgfqpoint{0.557cm}{1.47cm}}{\pgfqpoint{0.542cm}{1.435cm}}{\pgfqpoint{0.542cm}{1.399cm}}
\pgfpathcurveto{\pgfqpoint{0.542cm}{1.363cm}}{\pgfqpoint{0.557cm}{1.328cm}}{\pgfqpoint{0.582cm}{1.302cm}}
\pgfpathcurveto{\pgfqpoint{0.608cm}{1.276cm}}{\pgfqpoint{0.643cm}{1.262cm}}{\pgfqpoint{0.679cm}{1.262cm}}
\pgfpathcurveto{\pgfqpoint{0.715cm}{1.262cm}}{\pgfqpoint{0.75cm}{1.276cm}}{\pgfqpoint{0.775cm}{1.302cm}}
\pgfpathcurveto{\pgfqpoint{0.801cm}{1.328cm}}{\pgfqpoint{0.815cm}{1.363cm}}{\pgfqpoint{0.815cm}{1.399cm}}
\pgfusepath{fill}
\pgfpathmoveto{\pgfqpoint{1.345cm}{1.371cm}}
\pgfpathcurveto{\pgfqpoint{1.345cm}{1.408cm}}{\pgfqpoint{1.331cm}{1.442cm}}{\pgfqpoint{1.305cm}{1.468cm}}
\pgfpathcurveto{\pgfqpoint{1.28cm}{1.494cm}}{\pgfqpoint{1.245cm}{1.508cm}}{\pgfqpoint{1.209cm}{1.508cm}}
\pgfpathcurveto{\pgfqpoint{1.172cm}{1.508cm}}{\pgfqpoint{1.138cm}{1.494cm}}{\pgfqpoint{1.112cm}{1.468cm}}
\pgfpathcurveto{\pgfqpoint{1.087cm}{1.442cm}}{\pgfqpoint{1.072cm}{1.408cm}}{\pgfqpoint{1.072cm}{1.371cm}}
\pgfpathcurveto{\pgfqpoint{1.072cm}{1.335cm}}{\pgfqpoint{1.087cm}{1.3cm}}{\pgfqpoint{1.112cm}{1.274cm}}
\pgfpathcurveto{\pgfqpoint{1.138cm}{1.249cm}}{\pgfqpoint{1.172cm}{1.234cm}}{\pgfqpoint{1.209cm}{1.234cm}}
\pgfpathcurveto{\pgfqpoint{1.245cm}{1.234cm}}{\pgfqpoint{1.28cm}{1.249cm}}{\pgfqpoint{1.305cm}{1.274cm}}
\pgfpathcurveto{\pgfqpoint{1.331cm}{1.3cm}}{\pgfqpoint{1.345cm}{1.335cm}}{\pgfqpoint{1.345cm}{1.371cm}}
\pgfusepath{fill}
\begin{pgfscope}
\pgfsetdash{}{0cm}
\pgfsetlinewidth{0.818mm}
\pgfsetroundcap
\pgfsetmiterlimit{4.0}
\pgfpathmoveto{\pgfqpoint{0.682cm}{0.671cm}}
\pgfpathlineto{\pgfqpoint{0.682cm}{0.042cm}}
\pgfusepath{stroke}
\end{pgfscope}
\end{pgfscope}
\end{pgfscope}
\end{pgfscope}
\end{tikzpicture}}} + \zeta, \qquad \zeta = - \LL^{- 1} [ 3\lambda
   ( \UU_{>} \llbracket X^2 \rrbracket ) \succ Y ] + \phi ,
\]
the operator $C$ is the continuum analog of the commutator $C_{\varepsilon}$ defined in \eqref{eq:ce}, the localizer $\UU_{>}$ is given by the constant $L_{0}$ from Lemma \ref{lem:Y1} and $B(\cdot)$ (appearing also in the limit $Z$, cf. \eqref{eq:def-Z}) is the uniform  limit of $b_{M,\varepsilon}-\tilde{b}_{M,\varepsilon}(\cdot)$ on $[\tau,T]$.
Let us denote $H (\varphi, X) \assign \llbracket \varphi^3 \rrbracket - \llbracket
X^3 \rrbracket$.

\

Remark that our uniform bounds remain valid for the limiting measure $\mu$.
As a consequence we obtain the following result.

\begin{lemma}
  \label{lemma:IF}Let $F : \mathcal{S}' (\mathbb{R}^3) \rightarrow \mathbb{R}$
  be a cylinder function such that
  \[ | F (\varphi) | + \| \mathD F (\varphi) \|_{B_{\infty, \infty}^{1 + 3
     \kappa} (\rho^{- 4 - \sigma})} \leqslant C_F \| \varphi \|_{H^{- 1 / 2 -2
     \kappa} (\rho^2)}^n \]
  for some $n \in \mathbb{N}$. Let $\mu$ be an accumulation point of the
  sequence of laws of $(\mathcal{E}^{\varepsilon} \varphi_{M, \varepsilon},
  \mathcal{E}^{\varepsilon} X_{M, \varepsilon})$. Then  (along a
  subsequence)  $\mathcal{E}^{\varepsilon} \mathcal{J}_{M, \varepsilon}
  (F) \rightarrow \mathcal{J}_{\mu} (F)$ in $\mathcal{S}' (\mathbb{R}^d)$,
  where $\mathcal{J}_{\mu} (F)$ is given by
  \[ \mathcal{J}_{\mu} (F) =\mathbb{E}_{\mu} \left[ \int_{\mathbb{R}} h (t) F
     (\varphi (t)) \llbracket X^3 \rrbracket (t) \mathd t \right]
     +\mathbb{E}_{\mu} \left[ \int_{\mathbb{R}} h (t) F (\varphi (t)) H
     (\varphi, X) (t) \mathd t \right] \backassign \mathcal{J}^X_{\mu} (F)
     +\mathcal{J}^H_{\mu} (F), \]
for any function $h$ as above, which is understood as an equality of distributions and the expectation is in the weak sense. Moreover, we have the estimate
  \[ \| \mathcal{J}^X_{\mu} (F) \|_{\CC^{- 3 / 2 - \kappa} (\rho^{\sigma})} +
     \| \mathcal{J}^H_{\mu} (F) \|_{B_{1, 1}^{- 1 - 3 \kappa} (\rho^{4+\sigma})}
     \lesssim_{\mu, h} C_F \]
  where the implicit constant depends on $\mu, h$ but not on $F$.
\end{lemma}

\begin{proof}
  For any cylinder function $F$ satisfying the assumptions and since
  $\tmop{supp} h \in [\tau, T]$ we have the following estimate for arbitrary
  conjugate exponents $p, p' \in (1, \infty)$
  \[ \| \mathcal{J}^X_{\mu} (F) \|_{\CC^{- 3 / 2 - \kappa} (\rho^{\sigma})}
     \lesssim_h \mathbb{E}_{\mu} \left[ \| t \mapsto F (\varphi (t))
     \|_{W^{\kappa, 1}_T} \| \llbracket X^3 \rrbracket \|_{W^{- \kappa,
     \infty}_T \CC^{- 3 / 2 - \kappa} (\rho^{\sigma})} \right] \]
  \[ \lesssim (\mathbb{E}_{\mu} [\| t \mapsto F (\varphi (t)) \|_{W^{\kappa,
     1}_T}^p])^{1 / p}  \left( \mathbb{E}_{\mu} \left[ \| \llbracket X^3
     \rrbracket \|_{W^{- \kappa, \infty}_T \CC^{- 3 / 2 - \kappa}
     (\rho^{\sigma})}^{p'} \right] \right)^{1 / p'} \]
  \[ \lesssim (\mathbb{E}_{\mu} [\| t \mapsto F (\varphi (t)) \|_{W^{\kappa,
     1}_T}^p])^{1 / p} \lesssim \left( \int_{[0, T]^2} \frac{\mathbb{E}_{\mu}
     | F (\varphi (t)) - F (\varphi (s)) |^p}{| t - s |^{(1 + \kappa) p}}
     \mathd t \mathd s \right)^{1 / p} . \]
  Since for arbitrary conjugate exponents $q, q' \in (1, \infty)$
  \[ \mathbb{E}_{\mu} | F (\varphi (t)) - F (\varphi (s)) |^p \leqslant
     \int_0^1 \mathbb{E}_{\mu} | \langle \mathD F (\varphi (s) + \tau (\varphi
     (t) - \varphi (s))), \varphi (t) - \varphi (s) \rangle |^p \mathd \tau \]
  \[ \leqslant \int_0^1 \mathd \tau (\mathbb{E}_{\mu} \| \mathD F (\varphi (s)
     + \tau (\varphi (t) - \varphi (s))) \|^{p q'}_{B_{\infty, \infty}^{1 + 3
     \kappa} (\rho^{- 4 - \sigma})})^{1 / q'}  (\mathbb{E}_{\mu} \| \varphi
     (t) - \varphi (s) \|^{p q}_{B_{1, 1}^{- 1 - 3 \kappa} (\rho^{4 +
     \sigma})})^{1 / q} \]
  \[ \lesssim C^p_F (\mathbb{E}_{\mu} \| \varphi (0) \|_{H^{- 1 / 2 -2 \kappa}
     (\rho^2)}^{n p q'})^{1 / q'} (\mathbb{E}_{\mu} \| \varphi (t) - \varphi
     (s) \|^{p q}_{B_{1, 1}^{- 1 - 3 \kappa} (\rho^{4 + \sigma})})^{1 / q}, \]
  we obtain due to Theorem~\ref{thm:tight} that
  \[ \| \mathcal{J}^X_{\mu} (F) \|_{\CC^{- 3 / 2 - \kappa} (\rho^{\sigma})}
     \lesssim C_F \left( \int_{[0, T]^2} \frac{\mathbb{E}_{\mu} \| \varphi (t)
     - \varphi (s) \|^{p q}_{B_{1, 1}^{- 1 - 3 \kappa} (\rho^{4 +
     \sigma})}}{| t - s |^{(1 + \kappa) p q}} \mathd t \mathd s \right)^{1
     / (p q)} \]
  \[ \lesssim C_F (\mathbb{E}_{\mu} \| \varphi \|^{p q}_{W_T^{\alpha, p q}
     B_{1, 1}^{- 1 - 3 \kappa} (\rho^{4 + \sigma})})^{1 / (p q)}, \]
  where $\alpha = 1 + \kappa - 1 / (p q)$. Finally, choosing $p, q \in (1,
  \infty)$ sufficiently small and $\kappa \in (0, 1)$ appropriately, we may
  apply the Sobolev embedding $W_T^{\beta, 1} \subset W_T^{\alpha, p q}$
  together with the uniform bound from Theorem~\ref{thm:phitight} (which
  remains valid in the limit) to deduce
  \[ \| \mathcal{J}^X_{\mu} (F) \|_{\CC^{- 3 / 2 - \kappa} (\rho^{\sigma})}
     \lesssim C_F (\mathbb{E}_{\mu} \| \varphi \|^{p q}_{W_T^{\beta, 1} B_{1,
     1}^{- 1 - 3 \kappa} (\rho^{4 + \sigma})})^{1 / (p q)} \lesssim C_F . \]

  To show the second bound in the statement of the lemma, we use the fact that
  $\tmop{supp} h \subset [\tau, T]$ for some $0 < \tau < T < \infty$ to
  estimate
  \[ \| \mathcal{J}^H_{\mu} (F) \|_{B_{1, 1}^{- 1 - 3 \kappa} (\rho^{4 +
     \sigma})} \leqslant \mathbb{E}_{\mu} [\| t \mapsto F (\varphi (t))
     \|_{L^{\infty}_{\tau, T}} \| H (\varphi, X) \|_{L^1_T B_{1, 1}^{- 1 - 3
     \kappa} (\rho^{4 + \sigma})}] \]
  \[ \leqslant C_F (\mathbb{E}_{\mu} \| \varphi \|_{L^{\infty}_{\tau, T} H^{-
     1 / 2 - 2\kappa} (\rho^2)}^{2 n})^{1 / 2} (\mathbb{E}_{\mu} \| H (\varphi,
     X) \|^2_{L^1_T B_{1, 1}^{- 1 - 3 \kappa} (\rho^{4 + \sigma})})^{1 / 2}
     \lesssim C_F, \]
  where the last inequality follows from Theorem~\ref{thm:phitight} and the
  bounds in the proof of Proposition~\ref{prop:reg}.
\end{proof}

\

Heuristically we can think of $\mathcal{J}_{\mu} (F)$ as given by
\[ \mathcal{J}_{\mu} (F) \approx \int F (\varphi) \llbracket \varphi^3
   \rrbracket (0) \nu (\mathd \varphi) . \]
However, as we have seen above, this expression is purely formal since
$\llbracket \varphi^3 \rrbracket$ is only a space-time distribution with
respect to $\mu$ and therefore $\llbracket \varphi^3 \rrbracket (0)$ is not a
well defined random variable. One has to consider $F \mapsto \mathcal{J}_{\mu}
(F)$ as a linear functional on cylinder functions taking values in
$\mathcal{S}' (\mathbb{R}^3)$ and satisfying the above properties.
Lemma~\ref{lemma:IF} presents a concrete probabilistic representation based on
the stationary stochastic quantization dynamics of the $\Phi^{4_{}}_3$
measure.

\

Alternatively, the distribution $\mathcal{J}_{\mu} (F)$ can be characterized
in terms of $\varphi (0)$ without using the dynamics, in particular, in the
spirit of the operator product expansion as follows.

\begin{lemma}\label{lem:OPE}
  Let $F$ be a cylinder function as in Lemma~\ref{lemma:IF} and $\nu$ the
  first marginal of $\mu$. Then there exists a sequence of constants $(c_N)_{N
  \in \mathbb{N}}$ tending to $\infty$ as $N \rightarrow \infty$ such that
  \[ \mathcal{J}_{\mu} (F) = \lim_{N \rightarrow \infty} \int F (\varphi)
     [(\Delta_{\leqslant N} \varphi)^3 - c_N (\Delta_{\leqslant N} \varphi)]
     \nu (\mathd \varphi) \]
     in the sense of distributions. Moreover, the renormalization constants are given by
     $$
     c_{N}=3\lambda\mathbb{E}\big[\llbracket (\Delta_{\leqslant N}X)^{2}\rrbracket(t,0)\big]-18\lambda^{2}\mathbb{E}\big[\llbracket (\Delta_{\leqslant N}X)^{2}\rrbracket\circ\Q^{-1}\llbracket (\Delta_{\leqslant N}X)^{2}\rrbracket(t,0)\big],
     $$
     for some $t\geqslant 0$, where
     $$
     \llbracket (\Delta_{\leqslant N}X)^{2}\rrbracket=(\Delta_{\leqslant N}X)^{2}-\mathbb{E}\big[\llbracket (\Delta_{\leqslant N}X)^{2}\rrbracket(t,0)\big].
     $$
\end{lemma}

\begin{proof}
  Let
  \[ \mathcal{J}_{\nu, N} (F) \assign \int F (\varphi) [(\Delta_{\leqslant N}
     \varphi)^3 - c_N (\Delta_{\leqslant N} \varphi)] \nu (\mathd \varphi) .
  \]
  Then by stationarity of $\varphi$ under $\mu$ we have for a function $h$
  satisfying the above properties
  \[ \mathcal{J}_{\nu, N} (F) =\mathbb{E}_{\mu} \left[ \int_{\mathbb{R}} h (t)
     F (\varphi (t)) [(\Delta_{\leqslant N} \varphi (t))^3 - c_N
     (\Delta_{\leqslant N} \varphi (t))] \mathd t \right] . \]
  At this point is not difficult to proceed as above and find suitable
  constants $(c_N)_{N \in \mathbb{N}}$ which deliver the appropriate
  renormalizations so that
  \[ [(\Delta_{\leqslant N} \varphi)^3 - c_N (\Delta_{\leqslant N} \varphi)]
     \rightarrow \llbracket \varphi^3 \rrbracket, \]
  and therefore, using the control of the moments, prove that
  \[ \mathcal{J}_{\nu, N} (F) \rightarrow \mathbb{E}_{\mu} \left[
     \int_{\mathbb{R}} h (t) F (\varphi (t)) \llbracket \varphi^3 \rrbracket
     (t) \mathd t \right] =\mathcal{J}_{\mu} (F) . \]
\end{proof}

\begin{remark}
  By the previous lemma it is now clear that $\mathcal{J}_{\mu}$ does not
  depends on $\mu$ but only on its first marginal $\nu$. So in the following
  we will write $\mathcal{J}_{\nu} \assign \mathcal{J}_{\mu}$ to stress this
  fact.
\end{remark}

Using these informations we can pass to the limit in the approximate
integration by parts formula {\eqref{eq:ibp1}} and obtain an integration by
parts formula for the $\Phi^4_3$ measure in the full space. This is the main
result of this section.

\begin{theorem}
  \label{thm:ibp}Any accumulation point $\nu$ of the sequence $(\nu_{M,
  \varepsilon} \circ (\mathcal{E}^{\varepsilon})^{- 1})_{M, \varepsilon}$
  satisfies
  \begin{equation}
    \int \mathD F (\varphi) \nu (\mathd \varphi) = 2 \int [(m^2 - \Delta)
    \varphi] F (\varphi) \nu (\mathd \varphi) + 2\lambda \mathcal{J}_{\nu} (F)
    \label{eq:IBP}
  \end{equation}
  in the sense of distributions.
\end{theorem}

When interpreted in terms of $n$-point correlation functions, the integration
by parts formula {\eqref{eq:IBP}} gives rise to the hierarchy of
Dyson--Schwinger equations for any limiting measure $\nu$.

\begin{corollary}
  \label{cor:SD}Let $n \in \mathbb{N}$. Any accumulation point $\nu$ of the
  sequence $(\nu_{M, \varepsilon} \circ (\mathcal{E}^{\varepsilon})^{- 1})_{M,
  \varepsilon}$ satisfies
  \[ \sum_{i = 1}^n \delta (x - x_i) \mathbb{E}_{\nu} [\varphi (x_1) \cdots
     \varphi (x_{i - 1}) \varphi (x_{i + 1}) \cdots \varphi (x_n)]
     =\mathbb{E}_{\nu} [[(m^2 - \Delta_x) \varphi (x)] \varphi (x_1) \cdots
     \varphi (x_n)] \]
  \[ - \lambda \lim_{N \rightarrow \infty} \mathbb{E}_{\nu} [\varphi (x_1) \cdots
     \varphi (x_n) ((\Delta_{\leqslant N} \varphi (x))^3 - c_N
     \Delta_{\leqslant N} \varphi (x))]_{} \]
  as an equality for distributions in $\mathcal{S}' (\mathbb{R}^3)^{\otimes (n
  + 1)}$.
\end{corollary}

In particular, this allow to express the (space-homogeneous) two-point
function $S_{2}^{\nu} (x - y) \assign \mathbb{E}_{\nu} [\varphi (x) \varphi
(y)]$ of $\nu$ as the solution to
\[ \delta (x - y) = (m^2 - \Delta_x) S^{\nu}_{2} (x - y) -\lambda \lim_{N \rightarrow
   \infty} [((\mathbb{I} \otimes \Delta_{\leqslant N}^{\otimes 3}) S^{\nu}_{4})
   (y, x, x, x) - c_N (\Delta_{\leqslant N} S^{\nu}_{2}) (x - y)], \]
where the right hand side includes the  four point function $S^{\nu}_{4}
(x_1, \ldots, x_4) \assign \mathbb{E}_{\nu} [\varphi (x_1) \cdots \varphi
(x_4)].$

Finally, we observe that the above arguments also allow us to pass to the
limit in the stochastic quantization equation and to identify the continuum
dynamics. To be more precise, we use Skorokhod's representation theorem to
obtain a new probability space together with (not relabeled) processes
$(\varphi_{M, \varepsilon}, \mathbb{X}_{M, \varepsilon})$ defined on some
probability space and converging in the appropriate topology determined above
to some $(\varphi, \mathbb{X})$. We deduce the following result.

\begin{corollary}
The couple $(\varphi, \mathbb{X})$ solves the continuum stochastic quantization
equation
\[ \LL \varphi +\lambda \llbracket \varphi^3 \rrbracket = \xi
   \qquad \tmop{in} \qquad \mathcal{S}' (\mathbb{R}_+ \times \mathbb{R}^d), \]
where $\xi = \LL X$ and $\llbracket \varphi^3 \rrbracket$ is given by
{\eqref{eq:phi3}}.
\end{corollary}

\section{Fractional $\Phi^4_3$}
\label{sec:fractional}

In this section we discuss the extension of the results of this  paper to the \emph{fractional $\Phi^4_3$ model}, namely to the limit  of the following discrete Gibbs measures. Let $\gamma \in (0, 1)$ and set
{\small{\begin{equation}
  \mathd \nu_{M, \varepsilon}^{\gamma} \propto \exp \left\{ - 2 \varepsilon^d 
  \sum_{x \in \Lambda_{M, \varepsilon}} \left[ \frac{\lambda}{4} | \varphi_x
  |^4 + \frac{- 3 \lambda a_{M, \varepsilon} + 3 \lambda^2 b_{M, \varepsilon}
  + m^2}{2} | \varphi_x |^2 + \frac{1}{2} | (- \Delta_{\varepsilon})^{\gamma /
  2} \varphi_x |^2 \right] \right\}  \prod_{x \in \Lambda_{M, \varepsilon}}
  \hspace{-0.17em} \hspace{-0.17em} \mathd \varphi_x,
  \label{eq:fractional-gibbs}
\end{equation}}}
where $(- \Delta_{\varepsilon})^{\gamma}$ is the (discrete) fractional
Laplacian operator given through Fourier transform by
\[ \mathcal{F} ((- \Delta_{\varepsilon})^{\gamma} f) (k) = l_{\varepsilon}
   (k)^{\gamma}  \hat{f} (k), \]
with $l_{\varepsilon} (k) := \sum_{j = 1}^3 4 \sin^2 (\varepsilon \pi k_j) /
\varepsilon^2$. The kernel of the operator $(- \Delta_{\varepsilon})^{\gamma}$ on the lattice $(\varepsilon \mathbb{Z})^3$ has power-law decay in  space and therefore the above measure corresponds to a non-Gaussian unbounded-spin system with long-range interactions. Varying $\gamma$ at fixed space dimension allows to explore a range of super-renormalizable models which approach the critical dimension  as $\gamma$ is lowered. These and similar models have been considered in~\cite{brydges_non_gaussian_1998, MR2004988, MR2350436, slade_critical_2018, MR3874867} as rigorous ways to implement Wilson's and Fisher's $\varepsilon$-expansion idea, namely the study of critical models perturbatively in the distance to the critical dimension.

Let us first observe  that the measure $\nu_{M, \varepsilon}^{\gamma}$ is reflection positive. Albeit this result
seems to belong to the folklore of the mathematical physics community, we could
not find a clear reference to this fact and therefore we will give a sketch  of the proof. We
start from the observation that the fractional Laplacian generates a
reflection positive Gaussian measure. The proof we report below is due to
A.~Abdesselam (private communication). Recall that on $\Lambda_{M,\varepsilon}$ we define reflections $\theta^i$ with  $i=1,2,3$ and the reflection positivity as in Section~\ref{ss:OS2}. 
Below, the reflection positivity is always understood with respect to $\theta=\theta^1$. Of course,  similar considerations hold for the other directions as well.

\begin{theorem}
  Let $a > 0$, $\gamma \in
  (0, 1)$ and let  $\mu_{M, \varepsilon}^{\gamma}$  be the Gaussian measure on  $\Lambda_{M, \varepsilon}$ with
  covariance given by $(a - \Delta_{\varepsilon})^{- \gamma}$. Then $\mu_{M, \varepsilon}^{\gamma}$  is reflection positive.
\end{theorem}

\begin{proof}
  Let $\rho > 0$ and let $K_{\gamma} (\rho) \assign \int_0^{\infty}
  \frac{\mathd t}{t^{\gamma} (t + \rho)}$, so that $K_{\gamma} (\rho) =
  \rho^{- \gamma} K_{\gamma} (1)$. As a consequence we have the formula (as
  Fourier multipliers)
  \[ (a - \Delta_{\varepsilon})^{- \gamma} = \frac{1}{K_{\gamma} (1)}
     \int_0^{\infty} (t + a - \Delta_{\varepsilon})^{- 1} \frac{\mathd
     t}{t^{\gamma}} . \]
  Now the Gaussian measure with covariance $(t + a - \Delta_{\varepsilon})^{-
  1}$ corresponds to a spin-spin nearest neighbors interaction and is well
  known to be reflection positive (see the discussion in Section~\ref{ss:OS2}). In
  particular,
  \[ \sum_{x, y \in \Lambda_{M, \varepsilon}} \overline{\theta f (x)} f (y) (t
     + a - \Delta_{\varepsilon})^{- 1} (x, y) \geqslant 0, \]
  for all $f : \Lambda_{M, \varepsilon} \rightarrow \mathbb{C}$ supported on
  $\Lambda_{M, \varepsilon}^+ = \{ x \in \Lambda_{M, \varepsilon} : 0 < x_1
   < M / 2 \}$. Taking the appropriate integral over $t$ we get
  \[ \sum_{x, y \in \Lambda_{M, \varepsilon}} \overline{\theta f (x)} f (y) (a
     - \Delta_{\varepsilon})^{- \gamma} (x, y) \geqslant 0. \]
  From this we can deduce that, for all cylinder functions $F$ supported on
  $\Lambda_{M, \varepsilon}^+$ we have
  \[ \mathbb{E} [\overline{\theta F (\phi)} F (\phi)] \geqslant 0, \]
  where $\phi$ is the Gaussian field with covariance $(a -
  \Delta_{\varepsilon})^{- \gamma}$. This follows from taking $F$ as a linear
  combination of exponentials and then using Schur-Hadamard product theorem to
  deduce positivity and finally concluding by a density argument (see e.g. \cite[Thm 6.2.2]{MR887102}).
\end{proof}

\begin{corollary}
  The fractional $\Phi^4_3$ measure~{\eqref{eq:fractional-gibbs}} on
  $\Lambda_{M, \varepsilon}$ is reflection positive.
\end{corollary}

\begin{proof}
  Take $a > 0$ and consider the measure
  \[ \nu_{M, \varepsilon}^{\gamma, a} (\mathd \phi) = \frac{1}{Z^{\gamma,
     a}_{M, \varepsilon}} \rho_{\Lambda_{M, \varepsilon}} (\phi) \mu_{M,
     \varepsilon}^{\gamma} (\mathd \phi), \]
  where $\mu_{M, \varepsilon}^{\gamma}$ is, as above, the Gaussian measure with
  covariance $(a - \Delta_{\varepsilon})^{- \gamma}$ and
  \[ \rho_{\Lambda_{M, \varepsilon}} (\varphi) \assign \exp \left\{ - 2
     \varepsilon^d  \sum_{x \in \Lambda_{M, \varepsilon}} \left[
     \frac{\lambda}{4} | \varphi_x |^4 + \frac{- 3 \lambda a_{M, \varepsilon}
     + 3 \lambda^2 b_{M, \varepsilon} + m^2}{2} | \varphi_x |^2 \right]
     \right\} . \]
  Note that $\rho_{\Lambda_{M, \varepsilon}} (\varphi) = \rho_{\Lambda_{M,
  \varepsilon}^+} (\varphi) (\theta \rho_{\Lambda_{M, \varepsilon}^+})
  (\varphi)$ and that we can write
  \[ \int \overline{\theta F (\phi)} F (\phi) \nu_{M, \varepsilon}^{\gamma, a}
     (\mathd \phi) = \frac{1}{Z^{\gamma, a}_{M, \varepsilon}} \int
     \overline{\theta F (\phi)} F (\phi) \rho_{\Lambda_{M, \varepsilon}}
     (\phi) \mu_{M, \varepsilon}^{\gamma, a} (\mathd \phi) \]
  \[ = \frac{1}{Z^{\gamma, a}_{M, \varepsilon}} \int \overline{\theta
     (\rho_{\Lambda_{M, \varepsilon}^+} F) (\phi)} (\rho_{\Lambda_{M,
     \varepsilon}^+} F) (\phi) \mu_{M, \varepsilon}^{\gamma, a} (\mathd \phi)
     \geqslant 0, \]
  since we already proved that $\mu_{M, \varepsilon}^{\gamma, a}$ is
  reflection positive. Now, observe also that as $a \rightarrow 0$ the measures
  $(\nu_{M, \varepsilon}^{\gamma, a})_a$ converge weakly to $\nu_{M,
  \varepsilon}^{\gamma}$ and as a consequence we deduce that $\nu_{M,
  \varepsilon}^{\gamma}$ is reflection positive.
\end{proof}

The equilibrium stochastic dynamics associated to the measure $\nu_{M,
\varepsilon}^{\gamma}$ reads
\begin{equation}
  \LL _{\varepsilon}^{\gamma} \varphi_{M, \varepsilon} + \lambda \varphi_{M,
  \varepsilon}^3 + (- 3 \lambda a_{M, \varepsilon} + 3 \lambda^2 b_{M,
  \varepsilon}) \varphi_{M, \varepsilon} = \xi_{M, \varepsilon}, \qquad x \in
  \Lambda_{M, \varepsilon}, \label{eq:frac-moll}
\end{equation}
where $\LL _{\varepsilon}^{\gamma} = \partial_t + \Q_{\varepsilon}^{\gamma} $ and
$\Q_{\varepsilon}^{\gamma} = m^2 + (- \Delta_{\varepsilon})^{\gamma}$. 
We have to take into
account the different regularization properties of the fractional Laplacian, and the related modified space-time scaling
for the fractional heat equation. This  implies that the stochastic terms are of lower regularity. 
In particular,
 $X_{M,\varepsilon},\llbracket X_{M,\varepsilon}^{2}\rrbracket, \llbracket X_{M,\varepsilon}^{3}\rrbracket$ and $X_{M, \varepsilon}^{\!\resizebox{0.6em}{!}{
\begin{tikzpicture}
\pgfpathmoveto{\pgfqpoint{0cm}{-0.035cm}}
\pgfpathlineto{\pgfqpoint{1.376cm}{-0.035cm}}
\pgfpathlineto{\pgfqpoint{1.376cm}{1.552cm}}
\pgfpathlineto{\pgfqpoint{0cm}{1.552cm}}
\pgfpathclose
\pgfusepath{clip}
\begin{pgfscope}
\begin{pgfscope}
\pgfpathmoveto{\pgfqpoint{0cm}{-0.035cm}}
\pgfpathlineto{\pgfqpoint{1.376cm}{-0.035cm}}
\pgfpathlineto{\pgfqpoint{1.376cm}{1.552cm}}
\pgfpathlineto{\pgfqpoint{0cm}{1.552cm}}
\pgfpathclose
\pgfusepath{clip}
\begin{pgfscope}
\begin{pgfscope}
\pgfsetdash{}{0cm}
\pgfsetlinewidth{0.818mm}
\pgfsetroundcap
\pgfsetroundjoin
\pgfsetmiterlimit{7.0}
\definecolor{eps2pgf_color}{gray}{0}\pgfsetstrokecolor{eps2pgf_color}\pgfsetfillcolor{eps2pgf_color}
\pgfpathmoveto{\pgfqpoint{0.117cm}{1.421cm}}
\pgfpathlineto{\pgfqpoint{0.682cm}{0.671cm}}
\pgfpathlineto{\pgfqpoint{1.246cm}{1.421cm}}
\pgfusepath{stroke}
\end{pgfscope}
\definecolor{eps2pgf_color}{gray}{0}\pgfsetstrokecolor{eps2pgf_color}\pgfsetfillcolor{eps2pgf_color}
\pgfpathmoveto{\pgfqpoint{0.273cm}{1.395cm}}
\pgfpathcurveto{\pgfqpoint{0.273cm}{1.432cm}}{\pgfqpoint{0.259cm}{1.467cm}}{\pgfqpoint{0.233cm}{1.492cm}}
\pgfpathcurveto{\pgfqpoint{0.207cm}{1.518cm}}{\pgfqpoint{0.173cm}{1.532cm}}{\pgfqpoint{0.137cm}{1.532cm}}
\pgfpathcurveto{\pgfqpoint{0.1cm}{1.532cm}}{\pgfqpoint{0.066cm}{1.518cm}}{\pgfqpoint{0.04cm}{1.492cm}}
\pgfpathcurveto{\pgfqpoint{0.014cm}{1.467cm}}{\pgfqpoint{0cm}{1.432cm}}{\pgfqpoint{0cm}{1.395cm}}
\pgfpathcurveto{\pgfqpoint{0cm}{1.359cm}}{\pgfqpoint{0.014cm}{1.324cm}}{\pgfqpoint{0.04cm}{1.299cm}}
\pgfpathcurveto{\pgfqpoint{0.066cm}{1.273cm}}{\pgfqpoint{0.1cm}{1.258cm}}{\pgfqpoint{0.137cm}{1.258cm}}
\pgfpathcurveto{\pgfqpoint{0.173cm}{1.258cm}}{\pgfqpoint{0.207cm}{1.273cm}}{\pgfqpoint{0.233cm}{1.299cm}}
\pgfpathcurveto{\pgfqpoint{0.259cm}{1.324cm}}{\pgfqpoint{0.273cm}{1.359cm}}{\pgfqpoint{0.273cm}{1.395cm}}
\pgfusepath{fill}
\begin{pgfscope}
\pgfsetdash{}{0cm}
\pgfsetlinewidth{0.818mm}
\pgfsetmiterlimit{7.0}
\pgfpathmoveto{\pgfqpoint{0.682cm}{0.671cm}}
\pgfpathlineto{\pgfqpoint{0.679cm}{1.418cm}}
\pgfusepath{stroke}
\end{pgfscope}
\pgfpathmoveto{\pgfqpoint{0.815cm}{1.399cm}}
\pgfpathcurveto{\pgfqpoint{0.815cm}{1.435cm}}{\pgfqpoint{0.801cm}{1.47cm}}{\pgfqpoint{0.775cm}{1.496cm}}
\pgfpathcurveto{\pgfqpoint{0.75cm}{1.521cm}}{\pgfqpoint{0.715cm}{1.536cm}}{\pgfqpoint{0.679cm}{1.536cm}}
\pgfpathcurveto{\pgfqpoint{0.643cm}{1.536cm}}{\pgfqpoint{0.608cm}{1.521cm}}{\pgfqpoint{0.582cm}{1.496cm}}
\pgfpathcurveto{\pgfqpoint{0.557cm}{1.47cm}}{\pgfqpoint{0.542cm}{1.435cm}}{\pgfqpoint{0.542cm}{1.399cm}}
\pgfpathcurveto{\pgfqpoint{0.542cm}{1.363cm}}{\pgfqpoint{0.557cm}{1.328cm}}{\pgfqpoint{0.582cm}{1.302cm}}
\pgfpathcurveto{\pgfqpoint{0.608cm}{1.276cm}}{\pgfqpoint{0.643cm}{1.262cm}}{\pgfqpoint{0.679cm}{1.262cm}}
\pgfpathcurveto{\pgfqpoint{0.715cm}{1.262cm}}{\pgfqpoint{0.75cm}{1.276cm}}{\pgfqpoint{0.775cm}{1.302cm}}
\pgfpathcurveto{\pgfqpoint{0.801cm}{1.328cm}}{\pgfqpoint{0.815cm}{1.363cm}}{\pgfqpoint{0.815cm}{1.399cm}}
\pgfusepath{fill}
\pgfpathmoveto{\pgfqpoint{1.345cm}{1.371cm}}
\pgfpathcurveto{\pgfqpoint{1.345cm}{1.408cm}}{\pgfqpoint{1.331cm}{1.442cm}}{\pgfqpoint{1.305cm}{1.468cm}}
\pgfpathcurveto{\pgfqpoint{1.28cm}{1.494cm}}{\pgfqpoint{1.245cm}{1.508cm}}{\pgfqpoint{1.209cm}{1.508cm}}
\pgfpathcurveto{\pgfqpoint{1.172cm}{1.508cm}}{\pgfqpoint{1.138cm}{1.494cm}}{\pgfqpoint{1.112cm}{1.468cm}}
\pgfpathcurveto{\pgfqpoint{1.087cm}{1.442cm}}{\pgfqpoint{1.072cm}{1.408cm}}{\pgfqpoint{1.072cm}{1.371cm}}
\pgfpathcurveto{\pgfqpoint{1.072cm}{1.335cm}}{\pgfqpoint{1.087cm}{1.3cm}}{\pgfqpoint{1.112cm}{1.274cm}}
\pgfpathcurveto{\pgfqpoint{1.138cm}{1.249cm}}{\pgfqpoint{1.172cm}{1.234cm}}{\pgfqpoint{1.209cm}{1.234cm}}
\pgfpathcurveto{\pgfqpoint{1.245cm}{1.234cm}}{\pgfqpoint{1.28cm}{1.249cm}}{\pgfqpoint{1.305cm}{1.274cm}}
\pgfpathcurveto{\pgfqpoint{1.331cm}{1.3cm}}{\pgfqpoint{1.345cm}{1.335cm}}{\pgfqpoint{1.345cm}{1.371cm}}
\pgfusepath{fill}
\begin{pgfscope}
\pgfsetdash{}{0cm}
\pgfsetlinewidth{0.818mm}
\pgfsetroundcap
\pgfsetmiterlimit{4.0}
\pgfpathmoveto{\pgfqpoint{0.682cm}{0.671cm}}
\pgfpathlineto{\pgfqpoint{0.682cm}{0.042cm}}
\pgfusepath{stroke}
\end{pgfscope}
\end{pgfscope}
\end{pgfscope}
\end{pgfscope}
\end{tikzpicture}}}$ have respectively the spatial regularities  $(2\gamma-3)/2-$,  $(2\gamma-3)-$,  $3(2\gamma-3)/2-$, $(10\gamma-9)/2-$. 
 It is clear that using only the first order paracontrolled expansion developed in this  paper 
it is not possible to cover the full range of $\gamma$ for which the model is still  
subcritical (i.e. super-renormalizable). From eq.~\eqref{eq:frac-moll} one can readily compute that criticality in three-dimensions is reached when $\gamma=3/4$ at  which point the term $\llbracket X_{M,\varepsilon}^{2}\rrbracket$ scales like the fractional Laplacian. 

For large enough values of $\gamma \in (3/4,1)$ the analysis proceeds exactly in the case $\gamma = 1$. 
 Consequently  $Y_{M,\varepsilon}$ will also be of regularity $(10\gamma-9)/2-$ (cf. Lemma~\ref{lem:Y1}).  Since based on \eqref{eq:17}, $\phi_{\varepsilon}$ will have regularity $(4\gamma-3)-$, the various commutators 
      $D_{\rho^4, \varepsilon} (\phi_{M,\varepsilon}, - 3\lambda \llbracket
     X_{M,\varepsilon}^2 \rrbracket, \phi_{M,\varepsilon})$, $\langle \rho^4
     \phi_{M,\varepsilon}, \tilde{C}_{\varepsilon} (\phi_{M,\varepsilon}, 3\lambda
     \llbracket X_{M,\varepsilon}^2 \rrbracket, 3\lambda \llbracket X_{M,\varepsilon}^2
     \rrbracket) \rangle_{\varepsilon}$,
     and
  $ D_{\rho^4, \varepsilon} ( \phi_{M,\varepsilon}, 3\lambda \llbracket
     X_{M,\varepsilon}^2 \rrbracket, (\Q_{\varepsilon}^{\gamma})^{- 1} [3\lambda \llbracket
     X_{M,\varepsilon}^2 \rrbracket \succ \phi_{M,\varepsilon}] )$
 will be under control as  soon as 
 $
 (8\gamma-6)+(2\gamma-3)=10\gamma-9>0
 $
 namely when $\gamma > 9/10$. However, the term $Z_{M,\varepsilon}$ now has the regularity of the tree $X_{M,\varepsilon}^{\!\resizebox{!}{.8em}{
\begin{tikzpicture}
\pgfpathmoveto{\pgfqpoint{0cm}{-0.035cm}}
\pgfpathlineto{\pgfqpoint{1.976cm}{-0.035cm}}
\pgfpathlineto{\pgfqpoint{1.976cm}{1.94cm}}
\pgfpathlineto{\pgfqpoint{0cm}{1.94cm}}
\pgfpathclose
\pgfusepath{clip}
\begin{pgfscope}
\begin{pgfscope}
\pgfpathmoveto{\pgfqpoint{0cm}{-0.035cm}}
\pgfpathlineto{\pgfqpoint{1.976cm}{-0.035cm}}
\pgfpathlineto{\pgfqpoint{1.976cm}{1.94cm}}
\pgfpathlineto{\pgfqpoint{0cm}{1.94cm}}
\pgfpathclose
\pgfusepath{clip}
\begin{pgfscope}
\begin{pgfscope}
\pgfsetdash{}{0cm}
\pgfsetlinewidth{0.818mm}
\pgfsetroundcap
\pgfsetroundjoin
\pgfsetmiterlimit{7.0}
\definecolor{eps2pgf_color}{gray}{0}\pgfsetstrokecolor{eps2pgf_color}\pgfsetfillcolor{eps2pgf_color}
\pgfpathmoveto{\pgfqpoint{0.117cm}{1.815cm}}
\pgfpathlineto{\pgfqpoint{0.682cm}{1.065cm}}
\pgfpathlineto{\pgfqpoint{1.246cm}{1.815cm}}
\pgfusepath{stroke}
\end{pgfscope}
\definecolor{eps2pgf_color}{gray}{0}\pgfsetstrokecolor{eps2pgf_color}\pgfsetfillcolor{eps2pgf_color}
\pgfpathmoveto{\pgfqpoint{0.273cm}{1.789cm}}
\pgfpathcurveto{\pgfqpoint{0.273cm}{1.825cm}}{\pgfqpoint{0.259cm}{1.86cm}}{\pgfqpoint{0.233cm}{1.886cm}}
\pgfpathcurveto{\pgfqpoint{0.207cm}{1.912cm}}{\pgfqpoint{0.173cm}{1.926cm}}{\pgfqpoint{0.137cm}{1.926cm}}
\pgfpathcurveto{\pgfqpoint{0.1cm}{1.926cm}}{\pgfqpoint{0.066cm}{1.912cm}}{\pgfqpoint{0.04cm}{1.886cm}}
\pgfpathcurveto{\pgfqpoint{0.014cm}{1.86cm}}{\pgfqpoint{0cm}{1.825cm}}{\pgfqpoint{0cm}{1.789cm}}
\pgfpathcurveto{\pgfqpoint{0cm}{1.753cm}}{\pgfqpoint{0.014cm}{1.718cm}}{\pgfqpoint{0.04cm}{1.692cm}}
\pgfpathcurveto{\pgfqpoint{0.066cm}{1.667cm}}{\pgfqpoint{0.1cm}{1.652cm}}{\pgfqpoint{0.137cm}{1.652cm}}
\pgfpathcurveto{\pgfqpoint{0.173cm}{1.652cm}}{\pgfqpoint{0.207cm}{1.667cm}}{\pgfqpoint{0.233cm}{1.692cm}}
\pgfpathcurveto{\pgfqpoint{0.259cm}{1.718cm}}{\pgfqpoint{0.273cm}{1.753cm}}{\pgfqpoint{0.273cm}{1.789cm}}
\pgfusepath{fill}
\begin{pgfscope}
\pgfsetdash{}{0cm}
\pgfsetlinewidth{0.818mm}
\pgfsetmiterlimit{7.0}
\pgfpathmoveto{\pgfqpoint{0.682cm}{1.065cm}}
\pgfpathlineto{\pgfqpoint{0.679cm}{1.812cm}}
\pgfusepath{stroke}
\end{pgfscope}
\pgfpathmoveto{\pgfqpoint{0.815cm}{1.793cm}}
\pgfpathcurveto{\pgfqpoint{0.815cm}{1.829cm}}{\pgfqpoint{0.801cm}{1.864cm}}{\pgfqpoint{0.775cm}{1.89cm}}
\pgfpathcurveto{\pgfqpoint{0.75cm}{1.915cm}}{\pgfqpoint{0.715cm}{1.93cm}}{\pgfqpoint{0.679cm}{1.93cm}}
\pgfpathcurveto{\pgfqpoint{0.643cm}{1.93cm}}{\pgfqpoint{0.608cm}{1.915cm}}{\pgfqpoint{0.582cm}{1.89cm}}
\pgfpathcurveto{\pgfqpoint{0.557cm}{1.864cm}}{\pgfqpoint{0.542cm}{1.829cm}}{\pgfqpoint{0.542cm}{1.793cm}}
\pgfpathcurveto{\pgfqpoint{0.542cm}{1.756cm}}{\pgfqpoint{0.557cm}{1.722cm}}{\pgfqpoint{0.582cm}{1.696cm}}
\pgfpathcurveto{\pgfqpoint{0.608cm}{1.67cm}}{\pgfqpoint{0.643cm}{1.656cm}}{\pgfqpoint{0.679cm}{1.656cm}}
\pgfpathcurveto{\pgfqpoint{0.715cm}{1.656cm}}{\pgfqpoint{0.75cm}{1.67cm}}{\pgfqpoint{0.775cm}{1.696cm}}
\pgfpathcurveto{\pgfqpoint{0.801cm}{1.722cm}}{\pgfqpoint{0.815cm}{1.756cm}}{\pgfqpoint{0.815cm}{1.793cm}}
\pgfusepath{fill}
\pgfpathmoveto{\pgfqpoint{1.345cm}{1.765cm}}
\pgfpathcurveto{\pgfqpoint{1.345cm}{1.801cm}}{\pgfqpoint{1.331cm}{1.836cm}}{\pgfqpoint{1.305cm}{1.862cm}}
\pgfpathcurveto{\pgfqpoint{1.28cm}{1.887cm}}{\pgfqpoint{1.245cm}{1.902cm}}{\pgfqpoint{1.209cm}{1.902cm}}
\pgfpathcurveto{\pgfqpoint{1.172cm}{1.902cm}}{\pgfqpoint{1.138cm}{1.887cm}}{\pgfqpoint{1.112cm}{1.862cm}}
\pgfpathcurveto{\pgfqpoint{1.087cm}{1.836cm}}{\pgfqpoint{1.072cm}{1.801cm}}{\pgfqpoint{1.072cm}{1.765cm}}
\pgfpathcurveto{\pgfqpoint{1.072cm}{1.728cm}}{\pgfqpoint{1.087cm}{1.694cm}}{\pgfqpoint{1.112cm}{1.668cm}}
\pgfpathcurveto{\pgfqpoint{1.138cm}{1.642cm}}{\pgfqpoint{1.172cm}{1.628cm}}{\pgfqpoint{1.209cm}{1.628cm}}
\pgfpathcurveto{\pgfqpoint{1.245cm}{1.628cm}}{\pgfqpoint{1.28cm}{1.642cm}}{\pgfqpoint{1.305cm}{1.668cm}}
\pgfpathcurveto{\pgfqpoint{1.331cm}{1.694cm}}{\pgfqpoint{1.345cm}{1.728cm}}{\pgfqpoint{1.345cm}{1.765cm}}
\pgfusepath{fill}
\begin{pgfscope}
\pgfsetdash{}{0cm}
\pgfsetlinewidth{0.818mm}
\pgfsetroundcap
\pgfsetroundjoin
\pgfsetmiterlimit{7.0}
\pgfpathmoveto{\pgfqpoint{0.682cm}{1.065cm}}
\pgfpathlineto{\pgfqpoint{1.246cm}{0.315cm}}
\pgfpathlineto{\pgfqpoint{1.811cm}{1.065cm}}
\pgfusepath{stroke}
\end{pgfscope}
\pgfpathmoveto{\pgfqpoint{1.948cm}{1.065cm}}
\pgfpathcurveto{\pgfqpoint{1.948cm}{1.101cm}}{\pgfqpoint{1.933cm}{1.136cm}}{\pgfqpoint{1.907cm}{1.162cm}}
\pgfpathcurveto{\pgfqpoint{1.882cm}{1.187cm}}{\pgfqpoint{1.847cm}{1.202cm}}{\pgfqpoint{1.811cm}{1.202cm}}
\pgfpathcurveto{\pgfqpoint{1.775cm}{1.202cm}}{\pgfqpoint{1.74cm}{1.187cm}}{\pgfqpoint{1.714cm}{1.162cm}}
\pgfpathcurveto{\pgfqpoint{1.689cm}{1.136cm}}{\pgfqpoint{1.674cm}{1.101cm}}{\pgfqpoint{1.674cm}{1.065cm}}
\pgfpathcurveto{\pgfqpoint{1.674cm}{1.029cm}}{\pgfqpoint{1.689cm}{0.994cm}}{\pgfqpoint{1.714cm}{0.968cm}}
\pgfpathcurveto{\pgfqpoint{1.74cm}{0.942cm}}{\pgfqpoint{1.775cm}{0.928cm}}{\pgfqpoint{1.811cm}{0.928cm}}
\pgfpathcurveto{\pgfqpoint{1.847cm}{0.928cm}}{\pgfqpoint{1.882cm}{0.942cm}}{\pgfqpoint{1.907cm}{0.968cm}}
\pgfpathcurveto{\pgfqpoint{1.933cm}{0.994cm}}{\pgfqpoint{1.948cm}{1.029cm}}{\pgfqpoint{1.948cm}{1.065cm}}
\pgfusepath{fill}
\begin{pgfscope}
\pgfsetdash{}{0cm}
\pgfsetlinewidth{0.818mm}
\pgfsetmiterlimit{7.0}
\pgfpathmoveto{\pgfqpoint{1.246cm}{0.315cm}}
\pgfpathlineto{\pgfqpoint{1.244cm}{1.061cm}}
\pgfusepath{stroke}
\end{pgfscope}
\pgfpathmoveto{\pgfqpoint{1.38cm}{1.065cm}}
\pgfpathcurveto{\pgfqpoint{1.38cm}{1.101cm}}{\pgfqpoint{1.366cm}{1.136cm}}{\pgfqpoint{1.34cm}{1.162cm}}
\pgfpathcurveto{\pgfqpoint{1.315cm}{1.187cm}}{\pgfqpoint{1.28cm}{1.202cm}}{\pgfqpoint{1.244cm}{1.202cm}}
\pgfpathcurveto{\pgfqpoint{1.207cm}{1.202cm}}{\pgfqpoint{1.173cm}{1.187cm}}{\pgfqpoint{1.147cm}{1.162cm}}
\pgfpathcurveto{\pgfqpoint{1.121cm}{1.136cm}}{\pgfqpoint{1.107cm}{1.101cm}}{\pgfqpoint{1.107cm}{1.065cm}}
\pgfpathcurveto{\pgfqpoint{1.107cm}{1.029cm}}{\pgfqpoint{1.121cm}{0.994cm}}{\pgfqpoint{1.147cm}{0.968cm}}
\pgfpathcurveto{\pgfqpoint{1.173cm}{0.942cm}}{\pgfqpoint{1.207cm}{0.928cm}}{\pgfqpoint{1.244cm}{0.928cm}}
\pgfpathcurveto{\pgfqpoint{1.28cm}{0.928cm}}{\pgfqpoint{1.315cm}{0.942cm}}{\pgfqpoint{1.34cm}{0.968cm}}
\pgfpathcurveto{\pgfqpoint{1.366cm}{0.994cm}}{\pgfqpoint{1.38cm}{1.029cm}}{\pgfqpoint{1.38cm}{1.065cm}}
\pgfusepath{fill}
\begin{pgfscope}
\pgfsetdash{}{0cm}
\pgfsetlinewidth{0.818mm}
\pgfsetmiterlimit{4.0}
\pgfpathmoveto{\pgfqpoint{1.383cm}{0.178cm}}
\pgfpathcurveto{\pgfqpoint{1.383cm}{0.214cm}}{\pgfqpoint{1.369cm}{0.249cm}}{\pgfqpoint{1.343cm}{0.275cm}}
\pgfpathcurveto{\pgfqpoint{1.317cm}{0.3cm}}{\pgfqpoint{1.283cm}{0.315cm}}{\pgfqpoint{1.246cm}{0.315cm}}
\pgfpathcurveto{\pgfqpoint{1.21cm}{0.315cm}}{\pgfqpoint{1.175cm}{0.3cm}}{\pgfqpoint{1.15cm}{0.275cm}}
\pgfpathcurveto{\pgfqpoint{1.124cm}{0.249cm}}{\pgfqpoint{1.11cm}{0.214cm}}{\pgfqpoint{1.11cm}{0.178cm}}
\pgfpathcurveto{\pgfqpoint{1.11cm}{0.141cm}}{\pgfqpoint{1.124cm}{0.107cm}}{\pgfqpoint{1.15cm}{0.081cm}}
\pgfpathcurveto{\pgfqpoint{1.175cm}{0.055cm}}{\pgfqpoint{1.21cm}{0.041cm}}{\pgfqpoint{1.246cm}{0.041cm}}
\pgfpathcurveto{\pgfqpoint{1.283cm}{0.041cm}}{\pgfqpoint{1.317cm}{0.055cm}}{\pgfqpoint{1.343cm}{0.081cm}}
\pgfpathcurveto{\pgfqpoint{1.369cm}{0.107cm}}{\pgfqpoint{1.383cm}{0.141cm}}{\pgfqpoint{1.383cm}{0.178cm}}
\pgfusepath{stroke}
\end{pgfscope}
\end{pgfscope}
\end{pgfscope}
\end{pgfscope}
\end{tikzpicture}}}$ namely $(14\gamma-15)/2-$ and therefore in order to control $\langle\phi_{M,\varepsilon},Z_{M,\varepsilon}\rangle$ we must require $ \gamma>21/22$. 
 In this case
   the fractional energy estimate of Theorem~\ref{th:energy-estimate-int} carries through and provides a priori estimates for $\psi_{M,\varepsilon}$ in weighted $H^{\gamma}$ and as a consequence a similar estimate holds for $\zeta_{M,\varepsilon}$ in the same space. The proof of the stretched exponential integrability works as well but the exponent becomes worse due to the limited regularity of the stochastic terms. Moreover, the improved tightness in Section~\ref{s:reg} remains unchanged and yields the corresponding regularity.
Therefore, mutatis mutandis we conclude the following results.

\begin{theorem}
  \label{th:main-frac}Let $\gamma \in (21/22,1)$. There exists a choice of the sequence $(a_{M, \varepsilon},
  b_{M, \varepsilon})_{M, \varepsilon}$ such that for any $\lambda > 0$ and
  $m^2 \in \mathbb{R}$, the family of measures $(\nu^\gamma_{M, \varepsilon})_{M,
  \varepsilon}$ appropriately extended to $\mathcal{S}' (\mathbb{R}^3)$ is tight. 
  All the consequences stated  in Theorem~\ref{th:main} carry on to  
  every accumulation point $\nu$ of this family of measures except from the fact that the exponential integrability holds for some $\upsilon\in(0,1)$ not necessarily of order $\kappa$. 
\end{theorem}

If $\gamma \leqslant 21/22$ an additional renormalization is needed to treat the divergence of
$$(\Q^{\gamma}_{\varepsilon})^{-1}\llbracket X_{M,\varepsilon}^{2}\rrbracket\circ X_{M,\varepsilon}^{\!\resizebox{!}{.8em}{
\begin{tikzpicture}
\pgfpathmoveto{\pgfqpoint{0cm}{-0.035cm}}
\pgfpathlineto{\pgfqpoint{1.976cm}{-0.035cm}}
\pgfpathlineto{\pgfqpoint{1.976cm}{1.94cm}}
\pgfpathlineto{\pgfqpoint{0cm}{1.94cm}}
\pgfpathclose
\pgfusepath{clip}
\begin{pgfscope}
\begin{pgfscope}
\pgfpathmoveto{\pgfqpoint{0cm}{-0.035cm}}
\pgfpathlineto{\pgfqpoint{1.976cm}{-0.035cm}}
\pgfpathlineto{\pgfqpoint{1.976cm}{1.94cm}}
\pgfpathlineto{\pgfqpoint{0cm}{1.94cm}}
\pgfpathclose
\pgfusepath{clip}
\begin{pgfscope}
\begin{pgfscope}
\pgfsetdash{}{0cm}
\pgfsetlinewidth{0.818mm}
\pgfsetroundcap
\pgfsetroundjoin
\pgfsetmiterlimit{7.0}
\definecolor{eps2pgf_color}{gray}{0}\pgfsetstrokecolor{eps2pgf_color}\pgfsetfillcolor{eps2pgf_color}
\pgfpathmoveto{\pgfqpoint{0.117cm}{1.815cm}}
\pgfpathlineto{\pgfqpoint{0.682cm}{1.065cm}}
\pgfpathlineto{\pgfqpoint{1.246cm}{1.815cm}}
\pgfusepath{stroke}
\end{pgfscope}
\definecolor{eps2pgf_color}{gray}{0}\pgfsetstrokecolor{eps2pgf_color}\pgfsetfillcolor{eps2pgf_color}
\pgfpathmoveto{\pgfqpoint{0.273cm}{1.789cm}}
\pgfpathcurveto{\pgfqpoint{0.273cm}{1.825cm}}{\pgfqpoint{0.259cm}{1.86cm}}{\pgfqpoint{0.233cm}{1.886cm}}
\pgfpathcurveto{\pgfqpoint{0.207cm}{1.912cm}}{\pgfqpoint{0.173cm}{1.926cm}}{\pgfqpoint{0.137cm}{1.926cm}}
\pgfpathcurveto{\pgfqpoint{0.1cm}{1.926cm}}{\pgfqpoint{0.066cm}{1.912cm}}{\pgfqpoint{0.04cm}{1.886cm}}
\pgfpathcurveto{\pgfqpoint{0.014cm}{1.86cm}}{\pgfqpoint{0cm}{1.825cm}}{\pgfqpoint{0cm}{1.789cm}}
\pgfpathcurveto{\pgfqpoint{0cm}{1.753cm}}{\pgfqpoint{0.014cm}{1.718cm}}{\pgfqpoint{0.04cm}{1.692cm}}
\pgfpathcurveto{\pgfqpoint{0.066cm}{1.667cm}}{\pgfqpoint{0.1cm}{1.652cm}}{\pgfqpoint{0.137cm}{1.652cm}}
\pgfpathcurveto{\pgfqpoint{0.173cm}{1.652cm}}{\pgfqpoint{0.207cm}{1.667cm}}{\pgfqpoint{0.233cm}{1.692cm}}
\pgfpathcurveto{\pgfqpoint{0.259cm}{1.718cm}}{\pgfqpoint{0.273cm}{1.753cm}}{\pgfqpoint{0.273cm}{1.789cm}}
\pgfusepath{fill}
\begin{pgfscope}
\pgfsetdash{}{0cm}
\pgfsetlinewidth{0.818mm}
\pgfsetmiterlimit{7.0}
\pgfpathmoveto{\pgfqpoint{0.682cm}{1.065cm}}
\pgfpathlineto{\pgfqpoint{0.679cm}{1.812cm}}
\pgfusepath{stroke}
\end{pgfscope}
\pgfpathmoveto{\pgfqpoint{0.815cm}{1.793cm}}
\pgfpathcurveto{\pgfqpoint{0.815cm}{1.829cm}}{\pgfqpoint{0.801cm}{1.864cm}}{\pgfqpoint{0.775cm}{1.89cm}}
\pgfpathcurveto{\pgfqpoint{0.75cm}{1.915cm}}{\pgfqpoint{0.715cm}{1.93cm}}{\pgfqpoint{0.679cm}{1.93cm}}
\pgfpathcurveto{\pgfqpoint{0.643cm}{1.93cm}}{\pgfqpoint{0.608cm}{1.915cm}}{\pgfqpoint{0.582cm}{1.89cm}}
\pgfpathcurveto{\pgfqpoint{0.557cm}{1.864cm}}{\pgfqpoint{0.542cm}{1.829cm}}{\pgfqpoint{0.542cm}{1.793cm}}
\pgfpathcurveto{\pgfqpoint{0.542cm}{1.756cm}}{\pgfqpoint{0.557cm}{1.722cm}}{\pgfqpoint{0.582cm}{1.696cm}}
\pgfpathcurveto{\pgfqpoint{0.608cm}{1.67cm}}{\pgfqpoint{0.643cm}{1.656cm}}{\pgfqpoint{0.679cm}{1.656cm}}
\pgfpathcurveto{\pgfqpoint{0.715cm}{1.656cm}}{\pgfqpoint{0.75cm}{1.67cm}}{\pgfqpoint{0.775cm}{1.696cm}}
\pgfpathcurveto{\pgfqpoint{0.801cm}{1.722cm}}{\pgfqpoint{0.815cm}{1.756cm}}{\pgfqpoint{0.815cm}{1.793cm}}
\pgfusepath{fill}
\pgfpathmoveto{\pgfqpoint{1.345cm}{1.765cm}}
\pgfpathcurveto{\pgfqpoint{1.345cm}{1.801cm}}{\pgfqpoint{1.331cm}{1.836cm}}{\pgfqpoint{1.305cm}{1.862cm}}
\pgfpathcurveto{\pgfqpoint{1.28cm}{1.887cm}}{\pgfqpoint{1.245cm}{1.902cm}}{\pgfqpoint{1.209cm}{1.902cm}}
\pgfpathcurveto{\pgfqpoint{1.172cm}{1.902cm}}{\pgfqpoint{1.138cm}{1.887cm}}{\pgfqpoint{1.112cm}{1.862cm}}
\pgfpathcurveto{\pgfqpoint{1.087cm}{1.836cm}}{\pgfqpoint{1.072cm}{1.801cm}}{\pgfqpoint{1.072cm}{1.765cm}}
\pgfpathcurveto{\pgfqpoint{1.072cm}{1.728cm}}{\pgfqpoint{1.087cm}{1.694cm}}{\pgfqpoint{1.112cm}{1.668cm}}
\pgfpathcurveto{\pgfqpoint{1.138cm}{1.642cm}}{\pgfqpoint{1.172cm}{1.628cm}}{\pgfqpoint{1.209cm}{1.628cm}}
\pgfpathcurveto{\pgfqpoint{1.245cm}{1.628cm}}{\pgfqpoint{1.28cm}{1.642cm}}{\pgfqpoint{1.305cm}{1.668cm}}
\pgfpathcurveto{\pgfqpoint{1.331cm}{1.694cm}}{\pgfqpoint{1.345cm}{1.728cm}}{\pgfqpoint{1.345cm}{1.765cm}}
\pgfusepath{fill}
\begin{pgfscope}
\pgfsetdash{}{0cm}
\pgfsetlinewidth{0.818mm}
\pgfsetroundcap
\pgfsetroundjoin
\pgfsetmiterlimit{7.0}
\pgfpathmoveto{\pgfqpoint{0.682cm}{1.065cm}}
\pgfpathlineto{\pgfqpoint{1.246cm}{0.315cm}}
\pgfpathlineto{\pgfqpoint{1.811cm}{1.065cm}}
\pgfusepath{stroke}
\end{pgfscope}
\pgfpathmoveto{\pgfqpoint{1.948cm}{1.065cm}}
\pgfpathcurveto{\pgfqpoint{1.948cm}{1.101cm}}{\pgfqpoint{1.933cm}{1.136cm}}{\pgfqpoint{1.907cm}{1.162cm}}
\pgfpathcurveto{\pgfqpoint{1.882cm}{1.187cm}}{\pgfqpoint{1.847cm}{1.202cm}}{\pgfqpoint{1.811cm}{1.202cm}}
\pgfpathcurveto{\pgfqpoint{1.775cm}{1.202cm}}{\pgfqpoint{1.74cm}{1.187cm}}{\pgfqpoint{1.714cm}{1.162cm}}
\pgfpathcurveto{\pgfqpoint{1.689cm}{1.136cm}}{\pgfqpoint{1.674cm}{1.101cm}}{\pgfqpoint{1.674cm}{1.065cm}}
\pgfpathcurveto{\pgfqpoint{1.674cm}{1.029cm}}{\pgfqpoint{1.689cm}{0.994cm}}{\pgfqpoint{1.714cm}{0.968cm}}
\pgfpathcurveto{\pgfqpoint{1.74cm}{0.942cm}}{\pgfqpoint{1.775cm}{0.928cm}}{\pgfqpoint{1.811cm}{0.928cm}}
\pgfpathcurveto{\pgfqpoint{1.847cm}{0.928cm}}{\pgfqpoint{1.882cm}{0.942cm}}{\pgfqpoint{1.907cm}{0.968cm}}
\pgfpathcurveto{\pgfqpoint{1.933cm}{0.994cm}}{\pgfqpoint{1.948cm}{1.029cm}}{\pgfqpoint{1.948cm}{1.065cm}}
\pgfusepath{fill}
\begin{pgfscope}
\pgfsetdash{}{0cm}
\pgfsetlinewidth{0.818mm}
\pgfsetmiterlimit{7.0}
\pgfpathmoveto{\pgfqpoint{1.246cm}{0.315cm}}
\pgfpathlineto{\pgfqpoint{1.244cm}{1.061cm}}
\pgfusepath{stroke}
\end{pgfscope}
\pgfpathmoveto{\pgfqpoint{1.38cm}{1.065cm}}
\pgfpathcurveto{\pgfqpoint{1.38cm}{1.101cm}}{\pgfqpoint{1.366cm}{1.136cm}}{\pgfqpoint{1.34cm}{1.162cm}}
\pgfpathcurveto{\pgfqpoint{1.315cm}{1.187cm}}{\pgfqpoint{1.28cm}{1.202cm}}{\pgfqpoint{1.244cm}{1.202cm}}
\pgfpathcurveto{\pgfqpoint{1.207cm}{1.202cm}}{\pgfqpoint{1.173cm}{1.187cm}}{\pgfqpoint{1.147cm}{1.162cm}}
\pgfpathcurveto{\pgfqpoint{1.121cm}{1.136cm}}{\pgfqpoint{1.107cm}{1.101cm}}{\pgfqpoint{1.107cm}{1.065cm}}
\pgfpathcurveto{\pgfqpoint{1.107cm}{1.029cm}}{\pgfqpoint{1.121cm}{0.994cm}}{\pgfqpoint{1.147cm}{0.968cm}}
\pgfpathcurveto{\pgfqpoint{1.173cm}{0.942cm}}{\pgfqpoint{1.207cm}{0.928cm}}{\pgfqpoint{1.244cm}{0.928cm}}
\pgfpathcurveto{\pgfqpoint{1.28cm}{0.928cm}}{\pgfqpoint{1.315cm}{0.942cm}}{\pgfqpoint{1.34cm}{0.968cm}}
\pgfpathcurveto{\pgfqpoint{1.366cm}{0.994cm}}{\pgfqpoint{1.38cm}{1.029cm}}{\pgfqpoint{1.38cm}{1.065cm}}
\pgfusepath{fill}
\begin{pgfscope}
\pgfsetdash{}{0cm}
\pgfsetlinewidth{0.818mm}
\pgfsetmiterlimit{4.0}
\pgfpathmoveto{\pgfqpoint{1.383cm}{0.178cm}}
\pgfpathcurveto{\pgfqpoint{1.383cm}{0.214cm}}{\pgfqpoint{1.369cm}{0.249cm}}{\pgfqpoint{1.343cm}{0.275cm}}
\pgfpathcurveto{\pgfqpoint{1.317cm}{0.3cm}}{\pgfqpoint{1.283cm}{0.315cm}}{\pgfqpoint{1.246cm}{0.315cm}}
\pgfpathcurveto{\pgfqpoint{1.21cm}{0.315cm}}{\pgfqpoint{1.175cm}{0.3cm}}{\pgfqpoint{1.15cm}{0.275cm}}
\pgfpathcurveto{\pgfqpoint{1.124cm}{0.249cm}}{\pgfqpoint{1.11cm}{0.214cm}}{\pgfqpoint{1.11cm}{0.178cm}}
\pgfpathcurveto{\pgfqpoint{1.11cm}{0.141cm}}{\pgfqpoint{1.124cm}{0.107cm}}{\pgfqpoint{1.15cm}{0.081cm}}
\pgfpathcurveto{\pgfqpoint{1.175cm}{0.055cm}}{\pgfqpoint{1.21cm}{0.041cm}}{\pgfqpoint{1.246cm}{0.041cm}}
\pgfpathcurveto{\pgfqpoint{1.283cm}{0.041cm}}{\pgfqpoint{1.317cm}{0.055cm}}{\pgfqpoint{1.343cm}{0.081cm}}
\pgfpathcurveto{\pgfqpoint{1.369cm}{0.107cm}}{\pgfqpoint{1.383cm}{0.141cm}}{\pgfqpoint{1.383cm}{0.178cm}}
\pgfusepath{stroke}
\end{pgfscope}
\end{pgfscope}
\end{pgfscope}
\end{pgfscope}
\end{tikzpicture}}}.$$
In general, when $\gamma \in (3/4,21/22]$ more complex expansions and renormalizations are needed, either by exploiting the iterated commutator methods of Bailleul and Bernicot~\cite{bailleul_high_2016} or full fledged regularity structures~\cite{hairer_theory_2014,hairer_discretisations_2018}. While it is not clear that the local estimates of Moinat and Weber~\cite{moinat_space_time_2018} apply to the fractional Laplacian (which is a non-local operator), our energy method could be conceivably adapted to the regularity  structures framework. 
We prefer to leave these more substantial extensions to further  investigations.

\appendix
\section{Technical results}

\label{s:app}
In this section we present  auxiliary results needed in
the main body of the paper.

\subsection{Besov spaces}
First, we cover various properties of the
discrete weighted Besov spaces such as an equivalent formulation of the norms,
duality, interpolation, embeddings, bounds for powers of functions and a
weighted Young's inequality.

\begin{lemma}
  \label{lem:equiv2}Let $\alpha \in \mathbb{R}$, $p, q \in [1, \infty]$. Fix
  $n > | \alpha |$ and assume that $\rho$ is a weight such that
  \[ \| \rho \|_{B^{n + 1, \varepsilon}_{\infty, \infty} (\rho^{- 1})} + \|
     \rho^{- 1} \|_{B^{n + 1, \varepsilon}_{\infty, \infty} (\rho)} \lesssim 1
  \]
  uniformly in $\varepsilon$. Then
  \[ \| f \|_{B^{\alpha, \varepsilon}_{p, q} (\rho)} \sim \| \rho f
     \|_{B^{\alpha, \varepsilon}_{p, q}}, \]
  where the proportionality constant does not depend on $\varepsilon$.
\end{lemma}

\begin{proof}
  We write $\rho f = \rho \prec f + \rho \succcurlyeq f$ and estimate by
  paraproduct estimates
  \[ \| \rho \prec f \|_{B^{\alpha, \varepsilon}_{p, q}} = \| \rho \prec f
     \|_{B^{\alpha, \varepsilon}_{p, q} (\rho^{- 1} \rho)} \lesssim \| \rho
     \|_{L^{\infty, \varepsilon} (\rho^{- 1})} \| f \|_{B^{\alpha,
     \varepsilon}_{p, q} (\rho)} \lesssim \| f \|_{B^{\alpha, \varepsilon}_{p,
     q} (\rho)}, \]
  \[ \| \rho \succcurlyeq f \|_{B^{\alpha, \varepsilon}_{p, q}} = \| \rho
     \succcurlyeq f \|_{B^{\alpha, \varepsilon}_{p, q} (\rho^{- 1} \rho)}
     \lesssim \| f \|_{B^{\alpha, \varepsilon}_{p, \infty} (\rho)} \| \rho
     \|_{B^{n, \varepsilon}_{\infty, q} (\rho^{- 1})} \lesssim \| f
     \|_{B^{\alpha, \varepsilon}_{p, q} (\rho)} \| \rho \|_{B^{n + 1,
     \varepsilon}_{\infty, \infty} (\rho^{- 1})} \]
  \[ \lesssim \| f \|_{B^{\alpha, \varepsilon}_{p, q} (\rho)}, \]
  which implies one inequality. For the converse one, we write $f = \rho^{- 1}
  \prec (\rho f) + \rho^{- 1} \succcurlyeq (\rho f)$, and estimate
  \[ \| \rho^{- 1} \prec (\rho f) \|_{B^{\alpha, \varepsilon}_{p, q} (\rho)}
     \lesssim \| \rho^{- 1} \|_{L^{\infty, \varepsilon} (\rho)} \| \rho f
     \|_{B^{\alpha, \varepsilon}_{p, q}}, \]
  \[ \| \rho^{- 1} \succcurlyeq (\rho f) \|_{B^{\alpha, \varepsilon}_{p, q}
     (\rho)} \lesssim \| \rho f \|_{B^{\alpha, \varepsilon}_{p, \infty}} \|
     \rho^{- 1} \|_{B^{n, \varepsilon}_{\infty, q} (\rho)} \lesssim \| \rho f
     \|_{B^{\alpha, \varepsilon}_{p, q}} \| \rho^{- 1} \|_{B^{n + 1,
     \varepsilon}_{\infty, \infty} (\rho)} . \]
\end{proof}

\begin{lemma}
  \label{lem:dual2}Let $\alpha \in \mathbb{R}$, $p, p', q, q' \in [1, \infty]$
  such that $p, p'$ and $q, q'$ are conjugate exponents. Let $\rho$ be a
  weight as in Lemma~\ref{lem:equiv2}. Then
  \[ \langle f, g \rangle_{\varepsilon} \lesssim \| f \|_{B_{p, q}^{\alpha,
     \varepsilon} (\rho)} \| g \|_{B_{p', q'}^{- \alpha, \varepsilon}
     (\rho^{- 1})} \]
  with a proportionality constant independent of $\varepsilon$. Consequently,
  $B^{- \alpha, \varepsilon}_{p', q'} (\rho^{- 1}) \subset (B^{\alpha,
  \varepsilon}_{p, q} (\rho^{- 1}))^{\ast}$.
\end{lemma}

\begin{proof}
  In view of Lemma~\ref{lem:equiv2} it is sufficient to consider the
  unweighted case. Let $f \in B^{\alpha, \varepsilon}_{p, q}$ and $g \in B^{-
  \alpha, \varepsilon}_{p', q'}$. Then by Parseval's theorem and H{\"o}lder's
  inequality we have
  \[ \varepsilon^d \sum_{x \in \Lambda_{\varepsilon}} f (x) g (x) = \sum_{- 1
     \leqslant i, j \leqslant N - J} \varepsilon^d \sum_{x \in
     \Lambda_{\varepsilon}} \Delta_i^{\varepsilon} f (x)
     \Delta_j^{\varepsilon} g (x) \]
  \[ = \sum_{- 1 \leqslant i, j \leqslant N - J, i \sim j}
     \int_{\hat{\Lambda}_{\varepsilon}} \varphi_i (k) \mathcal{F} f (k)
     \varphi_j (k) \mathcal{F} g (k) \mathd k \]
  \[ = \sum_{- 1 \leqslant i, j \leqslant N - J, i \sim j} 2^{\alpha j} 2^{-
     \alpha j} \varepsilon^d \sum_{x \in \Lambda_{\varepsilon}}
     \Delta_i^{\varepsilon} f (x) \Delta_j^{\varepsilon} g (x) \lesssim \| f
     \|_{B_{p, q}^{\alpha, \varepsilon}} \| g \|_{B_{p', q'}^{- \alpha,
     \varepsilon}} . \]
\end{proof}

\begin{lemma}
  \label{lem:int}Let $\varepsilon \in \mathcal{A}$. Let $\alpha, \alpha_0,
  \alpha_1, \beta, \beta_0, \beta_1 \in \mathbb{R}$, $p, p_0, p_1, q, q_0, q_1
  \in [1, \infty]$ and $\theta \in [0, 1]$ such that
  \[ \alpha = \theta \alpha_0 + (1 - \theta) \alpha_1, \quad \beta = \theta
     \beta_0 + (1 - \theta) \beta_1, \quad \frac{1}{p} = \frac{\theta}{p_0} +
     \frac{1 - \theta}{p_1}, \quad \frac{1}{q} = \frac{\theta}{q_0} + \frac{1
     - \theta}{q_1} . \]
  Then
  \[ \| f \|_{B^{\alpha, \varepsilon}_{p, q} (\rho^{\beta})} \leqslant \| f
     \|^{\theta}_{B^{\alpha_0, \varepsilon}_{p_0, q_0} (\rho^{\beta_0})} \| f
     \|^{1 - \theta}_{B^{\alpha_1, \varepsilon}_{p_1, q_1} (\rho^{\beta_1})} .
  \]
\end{lemma}

\begin{proof}
  The proof is a consequence of H{\"o}lder's inequality. Let us show the claim
  for $p$, $p_0$, $p_1$, $q$, $q_0$, $q_1 \in [1, \infty)$ and $\varepsilon \in
  \mathcal{A} \setminus \{ 0 \}$. If some of the exponents $p, p_0, p_1, q,
  q_0, q_1$ are infinite or we are in the continuous setting, the proof
  follows by obvious modifications. We write
  \[ \| \rho^{\beta} \Delta_j^{\varepsilon} f \|_{L^{p, \varepsilon}}^p =
     \varepsilon^d \sum_{x \in \Lambda_{\varepsilon}} | \rho^{\beta}
     \Delta_j^{\varepsilon} f (x) |^p = \varepsilon^d \sum_{k \in
     \Lambda_{\varepsilon}} (\rho^{\theta \beta_0 p} | \Delta_j^{\varepsilon}
     f (x) |^{\theta p}) (\rho^{(1 - \theta) \beta_1 p} |
     \Delta_j^{\varepsilon} f (x) |^{(1 - \theta) p}) \]
  and apply H{\"o}lder's inequality to the conjugate exponents
  $\frac{p_0}{\theta p}$ and $\frac{p_1}{(1 - \theta) p}$ to obtain
  \[ \| \rho^{\beta} \Delta_j^{\varepsilon} f \|_{L^{p, \varepsilon}}^p
     \leqslant \left( \varepsilon^d \sum_{x \in \Lambda_{\varepsilon}}
     \rho^{\beta_0 p_0} | \Delta_j^{\varepsilon} f |^{p_0} \right)^{\theta p /
     p_0} \left( \varepsilon^d \sum_{x \in \Lambda_{\varepsilon}}
     \rho^{\beta_1 p_1} | \Delta_j^{\varepsilon} f |^{p_1} \right)^{(1 -
     \theta) p / p_1} \]
  \[ = \| \Delta_j^{\varepsilon} f \|^{\theta p}_{L^{p_0, \varepsilon}
     (\rho^{\beta_0})} \| \Delta_j^{\varepsilon} f \|^{(1 - \theta)
     p}_{L^{p_1, \varepsilon} (\rho^{\beta_1})} . \]
  Consequently,
  \[ \| f \|^q_{B^{\alpha, \varepsilon}_{p, q} (\rho^{\beta})} \leqslant
     \sum_{- 1 \leqslant j \leqslant N - J} 2^{\alpha k q} \| \rho^{\beta}
     \Delta_j^{\varepsilon} f \|_{L^{p, \varepsilon}}^q \]
  \[ \leqslant \sum_{- 1 \leqslant j \leqslant N - J} \left( 2^{\theta
     \alpha_0 k q} \| \Delta_j^{\varepsilon} f \|^{\theta q}_{L^{p_0,
     \varepsilon} (\rho^{\beta_0})} \right) \left( 2^{(1 - \theta) \alpha_1 k
     q} \| \Delta_j^{\varepsilon} f \|^{(1 - \theta) q}_{L^{p_1, \varepsilon}
     (\rho^{\beta_1})} \right) \]
  and by H{\"o}lder's inequality to the conjugate exponents $\frac{q_0}{\theta
  q}$ and $\frac{q_1}{(1 - \theta) q}$
  \[ \| f \|^q_{B^{\alpha, \varepsilon}_{p, q} (\rho^{\beta})} \]
  \[ \leqslant \left( \sum_{- 1 \leqslant j \leqslant N - J} 2^{\alpha_0 k
     q_0} \| \Delta_j^{\varepsilon} f \|^{q_0}_{L^{p_0, \varepsilon}
     (\rho^{\beta_0})} \right)^{\theta q / q_0} \left( \sum_{- 1 \leqslant j
     \leqslant N - J} 2^{\alpha_1 k q_1} \| \Delta_j^{\varepsilon} f
     \|^{q_1}_{L^{p_1, \varepsilon} (\rho^{\beta_1})} \right)^{(1 - \theta) q
     / q_1} \]
  \[ = \| f \|^{\theta q}_{B^{\alpha_0, \varepsilon}_{p_0, q_0}
     (\rho^{\beta_0})} \| f \|^{(1 - \theta) q}_{B^{\alpha_1,
     \varepsilon}_{p_1, q_1} (\rho^{\beta_1})} . \]
\end{proof}

We note that by our construction of the Littlewood--Paley projectors on
$\Lambda_{\varepsilon}$, in each of the cases $j = - 1$, $j \in \{ 0, \ldots,
N - J - 1 \}$ and $j = N - J$, there exists an $L^1$-kernel $\mathcal{K}$ such
that the Littlewood--Paley block $\Delta^{\varepsilon}_j f$ is given by a
convolution with $2^{j d} \mathcal{K} (2^j \cdummy)$.  For notational simplicity we omit the
dependence of $\mathcal{K}$ on the three cases above.

\begin{lemma}
  \label{lem:emb}Let $\varepsilon \in \mathcal{A}$ and let $\beta > 0$. Then
  \[ L^{2, \varepsilon} (\rho) = B^{0, \varepsilon}_{2, 2} (\rho), \qquad
     L^{4, \varepsilon} (\rho) \subset B^{0, \varepsilon}_{4, \infty} (\rho)
  \]
  and the proportional constants do not depend on $\varepsilon$.
\end{lemma}

\begin{proof}
  Due to Lemma~\ref{lem:equiv2} together with Parseval's equality we directly
  obtain the first claim. Consequently, by Young's inequality together with
  the fact that $\frac{\rho (y)}{\rho (x)} \lesssim \rho^{- 1} (x - y)$ (for a
  universal proportionality constant that depends only on $\rho$) we have that
  \[ \| f \|_{B^{0, \varepsilon}_{4, \infty} (\rho)} = \sup_{- 1 \leqslant
     j \leqslant N - J} \| \Delta_j^{\varepsilon} f \|_{L^{4, \varepsilon}
     (\rho)} = \sup_{- 1 \leqslant j \leqslant N - J} \| 2^{j d} \mathcal{K}
     (2^j \cdummy) \ast f \|_{L^{4, \varepsilon} (\rho)} \]
  \[ \lesssim \sup_{- 1 \leqslant j \leqslant N - J} \| 2^{j d} \mathcal{K}
     (2^j \cdummy) \|_{L^{1, \varepsilon} (\rho^{- 1})} \| f \|_{L^{4,
     \varepsilon} (\rho)} \lesssim \| f \|_{L^{4, \varepsilon} (\rho)} . \]
\end{proof}

\begin{lemma}
  \label{lem:grad}Let $\kappa \in (0, 1)$, $p \in [1, \infty]$ and let $\rho$
  be a polynomial weight
  \[ \| f \|_{B^{1 - \kappa, \varepsilon}_{p, p} (\rho)} \lesssim \| f
     \|_{B^{- \kappa, \varepsilon}_{p, p} (\rho)} + \| \nabla_{\varepsilon} f
     \|_{B^{- \kappa, \varepsilon}_{p, p} (\rho)}, \]
  where the proportionality constant does not depend on $\varepsilon$.
\end{lemma}

\begin{proof}
  Let $j \geqslant 0$. Let $K_j = K_{j, \varepsilon} = \mathcal{F}^{- 1}
  \varphi^{\varepsilon}_j$ and denote $\bar{K}_j = \bar{K}_{j, \varepsilon} =
  \sum_{i \sim j} K_{i, \varepsilon}$. Then 
  $\Delta^{\varepsilon}_j f = \bar{K}_j \ast \Delta^{\varepsilon}_j f$ and we
  write
  \[ \bar{K}_j \ast \Delta^{\varepsilon}_j f = (\tmop{Id} -
     \Delta_{\varepsilon})^{- 1} (\tmop{Id} - \Delta_{\varepsilon}) (\bar{K}_j
     \ast \Delta^{\varepsilon}_j f) \]
  \begin{equation}
    = (\tmop{Id} - \Delta_{\varepsilon})^{- 1} (\bar{K}_j \ast
    \Delta^{\varepsilon}_j f) - (\tmop{Id} - \Delta_{\varepsilon})^{- 1}
    \nabla^{\ast}_{\varepsilon} \nabla_{\varepsilon} (\bar{K}_j \ast
    \Delta^{\varepsilon}_j f) . \label{eq:16}
  \end{equation}
  For the second term we use translation invariance of
  $\nabla_{\varepsilon}$ to obtain
  \[ (\tmop{Id} - \Delta_{\varepsilon})^{- 1} \nabla^{\ast}_{\varepsilon}
     \nabla_{\varepsilon} (\bar{K}_j \ast \Delta^{\varepsilon}_j f) =
     ((\tmop{Id} - \Delta_{\varepsilon})^{- 1} \nabla^{\ast}_{\varepsilon}
     \bar{K}_j) \ast (\Delta^{\varepsilon}_j \nabla_{\varepsilon} f) , \]
  hence by Young inequality
  \[ \| ((\tmop{Id} - \Delta_{\varepsilon})^{- 1} \nabla^{\ast}_{\varepsilon}
     \bar{K}_j) \ast (\Delta^{\varepsilon}_j \nabla_{\varepsilon} f) \|_{L^{p,
     \varepsilon} (\rho)} \lesssim \| (\tmop{Id} - \Delta_{\varepsilon})^{- 1}
     \nabla^{\ast}_{\varepsilon} \bar{K}_j \|_{L^{1, \varepsilon} (\rho^{-
     1})} \| \Delta^{\varepsilon}_j \nabla_{\varepsilon} f \|_{L^{p,
     \varepsilon} (\rho)} \]
  The kernel $\mathcal{V}_{j, \ell} \assign (\tmop{Id} -
  \Delta_{\varepsilon})^{- 1} \nabla^{\ast}_{\varepsilon, \ell} \bar{K}_j$ is
  given by
  \[ \mathcal{V}_{j, \ell} (k) = \int_{\hat{\Lambda}_{\varepsilon}} e^{2 \pi
     ik \cdot x} \frac{\varepsilon^{- 1}  (1 - e^{- 2 \pi i \varepsilon
     x_{\ell} })}{1 + 2 \sum_{p = 1}^d \varepsilon^{- 2} \sin^2 (\pi i
     \varepsilon x_p)} \bar{\varphi}^{\varepsilon}_j (x) \mathd x \]
  where $\bar{\varphi}_j^{\varepsilon} = \sum_{i \sim j}
  \varphi^{\varepsilon}_i$. Now using $(1 - 2^{2 j} \Delta_x)^M e^{2 \pi ik
  \cdot x} = (1 + 2^{2 j} | 2 \pi k |^2)^M e^{2 \pi ik \cdot x}$ and
  integrating by parts $(1 - \Delta_x)^M$ we have
  \[ | (1 + 2^{2 j} | 2 \pi k |^2)^M \mathcal{V}_{j, \ell} (k) | \leqslant
     \int_{\hat{\Lambda}_{\varepsilon}} \left| (1 - 2^{2 j} \Delta_x)^M \left[
     \frac{\varepsilon^{- 1}  (1 - e^{- 2 \pi i \varepsilon x_{\ell} })}{1 + 2
     \sum_{p = 1}^d \varepsilon^{- 2} \sin^2 (\pi i \varepsilon x_p)}
     \bar{\varphi}^{\varepsilon}_j (x) \right] \right| \mathd x \]
  and it is possible to check that (using that $\varepsilon 2^j \lesssim 1$)
  \[ \left| (1 - 2^{2 j} \Delta_x)^M \left[ \frac{\varepsilon^{- 1}  (1 - e^{-
     2 \pi i \varepsilon x_{\ell} })}{1 + 2 \sum_{p = 1}^d \varepsilon^{- 2}
     \sin^2 (\pi i \varepsilon x_p)} \bar{\varphi}^{\varepsilon}_j (x) \right]
     \right| \lesssim 2^{- j} \mathbb{I}_{2^j \tilde{\mathcal{A}}} \]
  uniformly in $j$ where $\tilde{\mathcal{A}}$ is an annulus centered at the
  origin. Therefore
  \[ | \mathcal{V}_{j, \ell} (k) | \lesssim 2^{- j} 2^{d j} (1 + 2^{2 j} | 2
     \pi k |^2)^{- M} \]
  and from this is easy to deduce that $\| \mathcal{V}_{j, \ell} \|_{L^{1,
  \varepsilon} (\rho^{- 1})} \lesssim 2^{- j}$ uniformly in $j$ and
  $\varepsilon$.
  
  A similar computation applies to the first term in {\eqref{eq:16}} to obtain
  \[ \| (\tmop{Id} - \Delta_{\varepsilon})^{- 1} (\bar{K}_j \ast
     \Delta^{\varepsilon}_j f) \|_{L^{p, \varepsilon} (\rho)} \lesssim \|
     (\tmop{Id} - \Delta_{\varepsilon})^{- 1} \bar{K}_j \|_{L^{1, \varepsilon}
     (\rho^{- 1})} \| \Delta^{\varepsilon}_j f \|_{L^{p, \varepsilon} (\rho)}
     \lesssim 2^{- 2 j} \| \Delta^{\varepsilon}_j f \|_{L^{p, \varepsilon}
     (\rho)} \]
  and the proof is complete.
\end{proof}

\begin{lemma}
  \label{lem:15}Let $\varepsilon \in \mathcal{A}$ and let $\iota > 0$. Let
  $\rho$ be a weight such that $\rho^{\iota} \in L^{4, 0}$. Then
  \[ \| \rho^{1 + \iota} f \|_{L^{2, \varepsilon}} \lesssim \| \rho f
     \|_{L^{4, \varepsilon}}, \]
  where the proportionality constant does not depend on $\varepsilon$.
\end{lemma}

\begin{proof}
  By H{\"o}lder's inequality
  \[ \| \rho^{1 + \iota} f \|_{L^{2, \varepsilon}} \leqslant \| \rho^{\iota}
     \|_{L^{4, \varepsilon}} \| \rho f \|_{L^{4, \varepsilon}}, \]
  and since for $| x - y | \leqslant 1$ the quotient $\frac{\rho (x)}{\rho
  (y)}$ is uniformly bounded above and below, it follows from Lemma A.2
  {\cite{MP17}} that
  \[ \| \rho^{\iota} \|_{L^{4, \varepsilon}}^4 = \varepsilon^d \sum_{x \in
     \Lambda_{\varepsilon}} \rho^{4 \iota} (x) \lesssim \int_{\mathbb{R}^d}
     \rho^{4 \iota} (x) \mathd x < \infty, \]
  where the proportional constant only depends on $\rho$.
\end{proof}

\begin{lemma}
  \label{lem:mult}Let $\alpha > 0$. Let $\rho_1, \rho_2$ be weights. Then for
  every $\beta > 0$
  \[ \| f^2 \|_{B^{\alpha, \varepsilon}_{1, 1} (\rho_1 \rho_2)} \lesssim \| f
     \|_{L^{2, \varepsilon} (\rho_1)} \| f \|_{H^{\alpha + 2 \beta,
     \varepsilon} (\rho_2)}, \]
  \[ \| f^3 \|_{B^{\alpha, \varepsilon}_{1, 1} (\rho_1^2 \rho_2)} \lesssim \|
     f \|_{L^{4, \varepsilon} (\rho_1)}^2 \| f \|_{H^{\alpha + 2 \beta,
     \varepsilon} (\rho_2)}, \]
  where the proportionality constants do not depend on $\varepsilon$.
\end{lemma}

\begin{proof}
  Due to the paraproduct estimates and the embeddings of Besov spaces, we have
  for every $\beta > 0$
  \[ \| f^2 \|_{B^{\alpha, \varepsilon}_{1, 1} (\rho_1 \rho_2)} \lesssim \| f
     \|_{B_{2, \infty}^{- \beta, \varepsilon} (\rho_1)} \| f \|_{B_{2,
     1}^{\alpha + \beta, \varepsilon} (\rho_2)} \lesssim \| f \|_{B_{2, 2}^{-
     \beta, \varepsilon} (\rho_1)} \| f \|_{B_{2, 2}^{\alpha + 2 \beta,
     \varepsilon} (\rho_2)} \]
  \[ \lesssim \| f \|_{L^{2, \varepsilon} (\rho_1)} \| f \|_{H^{\alpha + 2
     \beta, \varepsilon} (\rho_2)} . \]
  For the cubic term, we write
  \[ \| f^3 \|_{B^{\alpha, \varepsilon}_{1, 1} (\rho_1^2 \rho_2)} \lesssim \|
     f \prec f^2 \|_{B^{\alpha, \varepsilon}_{1, 1} (\rho_1^2 \rho_2)} + \| f
     \succ f^2 \|_{B^{\alpha, \varepsilon}_{1, 1} (\rho_1^2 \rho_2)} + \| f
     \circ f^2 \|_{B^{\alpha, \varepsilon}_{1, 1} (\rho_1^2 \rho_2)} \]
  and estimate each term separately. The second and the third term can be
  estimated directly by
  \[ \| f \succ f^2 \|_{B^{\alpha, \varepsilon}_{1, 1} (\rho_1^2 \rho_2)} + \|
     f \circ f^2 \|_{B^{\alpha, \varepsilon}_{1, 1} (\rho_1^2 \rho_2)}
     \lesssim \| f^2 \|_{B^{- \beta, \varepsilon}_{2, \infty} (\rho_1^2)} \| f
     \|_{B^{\alpha + \beta, \varepsilon}_{2, 1} (\rho_2)} \]
  \[ \lesssim \| f^2 \|_{B^{- \beta, \varepsilon}_{2, 2} (\rho_1^2)} \| f
     \|_{B^{\alpha + 2 \beta, \varepsilon}_{2, 2} (\rho_2)} \lesssim \| f
     \|_{L^{4, \varepsilon} (\rho_1)}^2 \| f \|_{H^{\alpha + 2 \beta,
     \varepsilon} (\rho_2)} . \]
  For the remaining term, we have
  \[ \| f \prec f^2 \|_{B^{\alpha, \varepsilon}_{1, 1} (\rho_1^2 \rho_2)}
     \lesssim \| f \|_{B^{- \beta, \varepsilon}_{4, \infty} (\rho_1)} \| f^2
     \|_{B^{\alpha + \beta, \varepsilon}_{4 / 3, 1} (\rho_1 \rho_2)} \]
  where by the paraproduct estimates and Lemma~\ref{lem:emb}
  \[ \| f^2 \|_{B^{\alpha + \beta, \varepsilon}_{4 / 3, 1} (\rho_1 \rho_2)}
     \lesssim \| f \|_{B^{- \beta, \varepsilon}_{4, \infty} (\rho_1)} \| f
     \|_{B^{\alpha + 2 \beta, \varepsilon}_{2, 1} (\rho_2)} \lesssim \| f
     \|_{L^{4, \varepsilon} (\rho_1)} \| f \|_{H^{\alpha + 2 \beta,
     \varepsilon} (\rho_2)} \]
  which completes the proof.
\end{proof}

\begin{lemma}
  \label{lem:young}Let $\rho$ be a polynomial weight. Let $p, q, r \in [1,
  \infty]$ be such that $\frac{1}{r} + 1 = \frac{1}{p} + \frac{1}{q}$. Then
  \[ \| f \ast_{\varepsilon} g \|_{L^{r, \varepsilon} (\rho)} \lesssim \| f \|_{L^{p,
     \varepsilon} (\rho^{- 1})} \| g \|_{L^{q, \varepsilon} (\rho)}, \]
  \[ \| f \ast_{\varepsilon} g \|_{L^{r, 0} (\rho)} \lesssim \sup_{y \in \mathbb{R}^d} \|
     (\rho^{- 1} f) (y - \cdummy) \|_{L^{p, \varepsilon}}^{\frac{r - p}{r}} \|
     f \|^{\frac{p}{r}}_{L^{p, 0} (\rho^{- 1})} \| g \|_{L^{q,
     \varepsilon} (\rho)}, \]
     where  $\ast_{\varepsilon}$  denotes the convolution on 
$\Lambda_{\varepsilon}$ and the proportionality constants are independent of $\varepsilon$.
\end{lemma}

\begin{proof}
  We observe that for a polynomial weight of the form $\rho (x) = \langle x
  \rangle^{- \nu}$ for some $\nu \geqslant 0$, we have that $\rho (y)
  \lesssim \rho (x) \rho^{- 1} (x - y)$. Accordingly,
  \[ | f \ast g (y) \rho (y) | = \left| \varepsilon^d \sum_{x \in
     \Lambda_{\varepsilon}} f (y - x) g (x) \rho (y) \right| \lesssim
     \varepsilon^d \sum_{x \in \Lambda_{\varepsilon}} | \rho f (y - x) |
     \rho^{- 1} (x - y) | g (x) | \rho (x) \]
  hence the claim follows by (unweighted) Young's inequality. For the second
  bound, we write
  \[ | f \ast g (y) \rho (y) | \lesssim \varepsilon^d \sum_{x \in
     \Lambda_{\varepsilon}} (| (\rho^{- 1} f) (y - x) |^p | (\rho g) (x)
     |^q)^{\frac{1}{r}} | (\rho^{- 1} f) (y - x) |^{\frac{r - p}{r}} | (\rho
     g) (x) |^{\frac{r - q}{r}} \]
  and apply H{\"o}lder's inequality with exponents $r, \frac{r p}{r - p},
  \frac{r q}{r - q}$
  \[ | f \ast g (y) \rho (y) | \lesssim \left( \varepsilon^d \sum_{x \in
     \Lambda_{\varepsilon}} | (\rho^{- 1} f) (y - x) |^p | \rho g (x) |^q
     \right)^{\frac{1}{r}} \| (\rho^{- 1} f) (y - \cdummy) \|_{L^{p,
     \varepsilon}}^{\frac{r - p}{r}} \| \rho g \|_{L^{q,
     \varepsilon}}^{\frac{r - q}{r}} \]
  \[ \leqslant \left( \varepsilon^d \sum_{x \in \Lambda_{\varepsilon}} |
     (\rho^{- 1} f) (y - x) |^p | \rho g (x) |^q \right)^{\frac{1}{r}} \sup_{y
     \in \mathbb{R}^d} \| (\rho^{- 1} f) (y - \cdummy) \|_{L^{p,
     \varepsilon}}^{\frac{r - p}{r}} \| \rho g \|_{L^{q,
     \varepsilon}}^{\frac{r - q}{r}} . \]
  Finally, taking the $r$th power and integrating completes the proof.
\end{proof}

  \subsection{Localizers}
  \label{s:l1}
As the next step, we introduce another equivalent formulation of the weighted
Besov spaces $B^{\alpha, \varepsilon}_{\infty, \infty} (\rho)$ in terms of
suitable point evaluation of the Littlewood--Paley decomposition. First, for
$J \in \mathbb{N}_0$ such that $N - J \leqslant J_{\varepsilon}$, $\alpha \in
\mathbb{R}$ and $\varepsilon \in \mathcal{A}$ we define the Besov space
$b^{\alpha, \varepsilon}_{\infty, \infty} (\rho)$ of sequences $\lambda =
(\lambda_{j, m})_{- 1 \leqslant j \leqslant N - J, m \in \mathbb{Z}^d}$ by the
norm
\[ \| \lambda \|_{b^{\alpha, \varepsilon}_{\infty, \infty} (\rho)} \assign
   \sup_{- 1 \leqslant j \leqslant N - J} 2^{\alpha j} \sup_{m \in
   \mathbb{Z}^d} \rho (2^{- j - J} m) | \lambda_{j, m} | . \]
Note that we do not stress the dependence of $b^{\alpha, \varepsilon}_{\infty,
\infty} (\rho)$ on the parameter $J$ as in the sequel we only consider one
fixed $J$ for all $\varepsilon \in \mathcal{A}$ given by Lemma~\ref{lem:equiv}
below.
The next result shows  the desired equivalence.

\begin{lemma}
  \label{lem:equiv}Let $\alpha \in \mathbb{R}$, $\varepsilon \in \mathcal{A}$
  and let $\rho$ be a weight. There exists $J \in \mathbb{N}_0$ (independent
  of $\varepsilon$) with the following property: $f \in B^{\alpha,
  \varepsilon}_{\infty, \infty} (\rho)$ if and only if it is represented by
  $\lambda = (\lambda_{j, m})_{- 1 \leqslant j \leqslant N - J, m \in
  \mathbb{Z}^d} \in b^{\alpha, \varepsilon}_{\infty, \infty} (\rho)$ such that
  \begin{equation}
    \| f \|_{B_{\infty, \infty}^{\alpha, \varepsilon} (\rho)} \sim \| \lambda
    \|_{b^{\alpha, \varepsilon}_{\infty, \infty} (\rho)}, \label{eq:d3}
  \end{equation}
  where the proportionality constants do not depend on $\varepsilon$. In
  particular, given $f \in B^{\alpha, \varepsilon}_{\infty, \infty} (\rho)$
  the coefficients $\lambda$ are defined by
  \begin{equation}
    \lambda_{j, m} (f) \assign \Delta_j^{\varepsilon} f (2^{- j - J} m),
    \qquad - 1 \leqslant j \leqslant N - J, \hspace{1em} m \in \mathbb{Z}^d,
    \label{eq:d1}
  \end{equation}
  and given $\lambda \in b^{\alpha, \varepsilon}_{\infty, \infty} (\rho)$ the
  distribution $f$ is recovered via the formula
  \begin{equation}
    f = \sum_{- 1 \leqslant j \leqslant N - J} 
    \mathcal{F}^{- 1} (\mathcal{F}_{2^{- j - J} \mathbb{Z}^d} (\lambda_{j,
    \cdot})), \label{eq:d2}
  \end{equation}
  where $\mathcal{F}_{2^{- j - J} \mathbb{Z}^d}$ denotes the Fourier transform
  on the lattice $2^{- j - J} \mathbb{Z}^d$.
\end{lemma}

 \begin{proof} 
  Let us first discuss the decomposition {\eqref{eq:d2}}. We recall that if $f
  \in \mathcal{S}' (\Lambda_{\varepsilon})$ then $\mathcal{F} f = \sum_{- 1
  \leqslant j \leqslant N - J} \varphi^{\varepsilon}_j \mathcal{F} f$ where
  for $j < N - J$ the function $\varphi^{\varepsilon}_j \mathcal{F} f$ is
  supported in a ball of radius proportional to $2^j$. Let $j < N - J$ and let
  $B_j \subset \mathbb{R}^d$ be a cube centered at the origin with length
  $2^{j + J}$. We choose $J \in \mathbb{N}_0$ such that $\tmop{supp}
  \varphi^{\varepsilon}_j \subset B_j$. Next, we identify $B_j$ with $(2^{j +
  J} \mathbb{T})^d \subset (2^N \mathbb{T})^d$ and regard
  $\varphi^{\varepsilon}_j \mathcal{F} f$ as a periodic function on $(2^{j +
  J} \mathbb{T})^d$. Then using a Fourier series expansion we may write
  \[ (\varphi^{\varepsilon}_j \mathcal{F} f) (z) = 2^{(- j - J) d} \sum_{m
     \in \mathbb{Z}^d} \lambda_{j, m} (f) e^{- 2 \pi i 2^{- j - J} m \cdummy
     z} = \mathcal{F}_{2^{- j - J} \mathbb{Z}^d} (\lambda_{j, \cdot} (f)) (z) \]
  where
  \[ \lambda_{j, m} (f) \assign \int_{B_j} (\varphi^{\varepsilon}_j
     \mathcal{F} f) (y) e^{2 \pi i 2^{- j - J} m \cdummy y} \mathd y =
     \mathcal{F}^{- 1} (\varphi^{\varepsilon}_j \mathcal{F} f) (2^{- j - J} m)
     = \Delta_j^{\varepsilon} f (2^{- j - J} m) . \]
  If $j = N - J$ then by definition of $\varphi^{\varepsilon}_j$ we see that
  $\varphi^{\varepsilon}_j \mathcal{F} f$ is a periodic function on $(2^N
  \mathbb{T})^d$. Hence we obtain the same formula (since $- j - J = - N$)
  \[ \lambda_{j, m} (f) \assign \int_{(2^N \mathbb{T})^d}
     (\varphi^{\varepsilon}_j \mathcal{F} f) (y) e^{2 \pi i 2^{- j - J} m
     \cdummy y} \mathd y = \Delta_j^{\varepsilon} f (2^{- j - J} m) . \]
  Therefore, we have derived the decomposition {\eqref{eq:d2}} with
  coefficients given by {\eqref{eq:d1}}.
  
  It remains to establish the equivalence of norms {\eqref{eq:d3}}. One
  direction is immediate, namely, for every $N - J \leqslant J_{\varepsilon}$
  we have
  \[ \sup_{- 1 \leqslant j \leqslant N - J} 2^{\alpha j} \sup_{m \in
     \mathbb{Z}^d} \rho (2^{- j - J} m) | \lambda_{j, m} (f) | = \sup_{- 1
     \leqslant j \leqslant N - J} 2^{\alpha j} \sup_{m \in \mathbb{Z}^d} \rho
     (2^{- j - J} m) | \Delta^{\varepsilon}_j f (2^{- j - J} m) | \]
  \[ \leqslant \sup_{- 1 \leqslant j \leqslant N - J} 2^{\alpha j} \sup_{x \in
     \Lambda_{\varepsilon}} \rho (x) | \Delta^{\varepsilon}_j f (x) | . \]
  Conversely, if $x \in \Lambda_{\varepsilon}$ belongs to the cube of size
  $2^{- j - J}$ centered at $2^{- j - J} m$, we write
  \begin{equation}
    | \Delta^{\varepsilon}_j f (x) | \leqslant | \Delta^{\varepsilon}_j f (x)
    - \Delta^{\varepsilon}_j f (2^{- j - J} m) | + | \Delta^{\varepsilon}_j f
    (2^{- j - J} m) |, \label{eq:25}
  \end{equation}
  Now we shall multiply the above inequality by $\rho (x)$ and estimate. To
  this end, we recall that due to the admissibility condition for polynomial
  weights there exists $\nu \geqslant 0$ and $c_1 > 0$ (depending only on
  $\rho$) such that
  \[ \frac{\rho (x)}{\rho (z)} \lesssim \big( 1 + \big| \sqrt{d} 2^{- j - J
     - 1} \big|^2 \big)^{\nu / 2} \lesssim c_1 \quad \tmop{whenever} \quad
     | x - z | \leqslant \sqrt{d} 2^{- j - J - 1} . \]
  In addition, to estimate the first term in {\eqref{eq:25}}, we recall that
  for $- 1 \leqslant j < N - J$ the Fourier transform of
  $\Delta^{\varepsilon}_j f$ is supported in a ball of radius proportional to
  $2^j$ hence by a computation similar to Bernstein's lemma (since by our
  construction $| x - 2^{- j - J} m | \leqslant \sqrt{d} 2^{- j - J - 1}$)
  \[ \rho (x) | \Delta^{\varepsilon}_j f (x) - \Delta^{\varepsilon}_j f (2^{-
     j - J} m) | \leqslant c_2 2^{- J - 1} \| \Delta^{\varepsilon}_j f
     \|_{L^{\infty, \varepsilon} (\rho)}, \]
  for some universal constant $c_2 > 0$ independent of $f$ and $\varepsilon$.
  If $j = N - J$ then $\Lambda_{\varepsilon}$ coincides with the lattice $2^{-
  j - J} \mathbb{Z}^d$ and therefore we do not need to do anything.
  Consequently it follows from {\eqref{eq:25}} that
  \[ \| \Delta^{\varepsilon}_j f \|_{L^{\infty, \varepsilon} (\rho)} \leqslant
     c_2 2^{- J - 1} \| \Delta^{\varepsilon}_j f \|_{L^{\infty, \varepsilon}
     (\rho)} + c_1 \sup_{m \in \mathbb{Z}^d} \rho (2^{- j - J} m) |
     \Delta^{\varepsilon}_j f (2^{- j - J} m) | . \]
  Hence, making $J \in \mathbb{N}_0$ possibly larger such that $c_2 2^{- J -
  1} < 1$, we may absorb the first term on the right hand side into the left
  hand side and the claim follows.
  \end{proof}

\begin{remark}
  Throughout the paper, the parameter $J \in \mathbb{N}_0$ is fixed as in Lemma~\ref{lem:equiv}. Consequently, from the condition $0 \leqslant N - J$ we
  obtain the necessary lower bound $N_0$ for $N$, or alternatively the upper
  bound for $\varepsilon = 2^{- N} \leqslant 2^{- N_0}$ and defines the set
  $\mathcal{A}$. These parameters remain fixed for the rest of the paper.
\end{remark}

\begin{remark}
  \label{rem:3}Note that the formulas {\eqref{eq:d1}}, {\eqref{eq:d2}} depend
  on the chosen partition of unity $(\varphi_j)_{j \geqslant - 1}$ and our
  construction of the associated periodic partitions of unity on
  $\hat{\Lambda}_{\varepsilon}$ via $\eqref{eq:p1} .$
\end{remark}

It follows from the previous lemma that we may identify $f \in B^{\alpha,
\varepsilon}_{\infty, \infty} (\rho)$ with its coefficients $(\lambda_{j, m}
(f))_{- 1 \leqslant j \leqslant N - J, m \in \mathbb{Z}^d} \in b^{\alpha,
\varepsilon}_{\infty, \infty} (\rho)$. This consideration leads us to the
definition of localization operators needed for the analysis of the $\Phi^4_3$
model. Although the principle idea is similar to Section 2.3 in {\cite{GH18}},
we present a different definition of the localizers here. It is based on the
equivalent description of the Besov spaces from Lemma~\ref{lem:equiv} and is
better suited for the discrete setting.

Given $(L_k)_{k \geqslant - 1} \subset (0, \infty)$ and $f \in \mathcal{S}'
(\Lambda_{\varepsilon})$ we define
\[ \UU_{>}^{\varepsilon} f \assign \left( \lambda_{j, m} \left(
   \UU^{\varepsilon}_{>} f \right) \right)_{- 1 \leqslant j \leqslant N - J, m
   \in \mathbb{Z}^d}, \qquad \UU_{\leqslant}^{\varepsilon} f \assign \left(
   \lambda_{j, m} \left( \UU_{\leqslant}^{\varepsilon} f \right) \right)_{- 1
   \leqslant j \leqslant N - J, m \in \mathbb{Z}^d} \]
where
\[ \lambda_{j, m} \left( \UU_{>}^{\varepsilon} f \right) \assign \left\{
   \begin{array}{lll}
     \lambda_{j, m} (f), &  & \tmop{if} | m | \sim 2^k \tmop{and} j > L_k \text{ for some } k\in \{-1,0,1,\dots\},\\
     0, &  & \tmop{otherwise},
   \end{array} \right. \]
\[ \lambda_{j, m} \left( \UU_{\leqslant}^{\varepsilon} f \right) \assign
   \left\{ \begin{array}{lll}
     \lambda_{j, m} (f), &  & \tmop{if} | m | \sim 2^k \tmop{and} j \leqslant
     L_k \text{ for some } k\in \{-1,0,1,\dots\},\\
     0, &  & \tmop{otherwise} .
   \end{array} \right. \]
We observe that by definition $f = \UU_{>}^{\varepsilon} f +
\UU_{\leqslant}^{\varepsilon} f$ and the localizers $\UU_{>}^{\varepsilon},
\UU_{\leqslant}^{\varepsilon}$ will only depend on $\varepsilon$ through the
cut-off of the coefficients $\lambda$ (and consequently on the construction of
the partition of unity on $\hat{\Lambda}_{\varepsilon}$, cf. Remark
\ref{rem:3}), whereas the sequence $(L_k)_{k \geqslant - 1}$ will be chosen
uniformly for all $\varepsilon \in \mathcal{A}$.

\begin{lemma}
  \label{lem:loc}Let $\rho$ be a weight. Let $\alpha, \beta, \gamma \in
  \mathbb{R}$ and $a, b, c \in \mathbb{R}$ such that $\alpha < \beta <
  \gamma$, $a < b < c$ and $r \assign (b - a) / (\beta - \alpha) = (c - b) /
  (\gamma - \beta) > 0$. Let $L > 0$ be given. There exists a sequence
  $(L_k)_{k \geqslant - 1}$ defining the above localizers such that
  \[ \left\| \UU^{\varepsilon}_{>} f \right\|_{B^{\alpha,
     \varepsilon}_{\infty, \infty} (\rho^a)} \lesssim 2^{- (\beta - \alpha) L}
     \| f \|_{B^{\beta, \varepsilon}_{\infty, \infty} (\rho^b)}, \]
  \[ \left\| \UU^{\varepsilon}_{\leqslant} f \right\|_{B^{\gamma,
     \varepsilon}_{\infty, \infty} (\rho^c)} \lesssim 2^{(\gamma - \beta) L}
     \| f \|_{B^{\beta, \varepsilon}_{\infty, \infty} (\rho^b)}, \]
  where the proportionality constants do not depend on $\varepsilon \in
  \mathcal{A}$. Moreover, the sequence $(L_k)_{k \geqslant - 1}$ depends only
  on $L, \rho$ and the ratio $r$.
\end{lemma}

\begin{proof}
  Since $\alpha < \beta$ and $a < b$,  Lemma~\ref{lem:equiv} yields
  \[ \left\| \UU_{>}^{\varepsilon} f \right\|_{B^{\alpha,
     \varepsilon}_{\infty, \infty} (\rho^a)} \lesssim \sup_{- 1 \leqslant j
     \leqslant N - J} 2^{\alpha j} \sup_{m \in \mathbb{Z}^d} \rho^a (2^{- j -
     J} m) \left| \lambda_{j, m} \left( \UU^{\varepsilon}_{>} f \right)
     \right| \]
  \[ = \sup_{k \geqslant - 1} \sup_{m \sim 2^k, L_k < j \leqslant N - J}
     2^{(\alpha - \beta) j} \rho^{a - b} (2^{- j - J} m) 2^{\beta j} \rho^b
     (2^{- j - J} m) | \lambda_{j, m} (f) | \]
  \[ \lesssim \| f \|_{B^{\beta, \varepsilon}_{\infty, \infty} (\rho^b)}
     \sup_{k \geqslant - 1} \sup_{m \sim 2^k, L_k < j \leqslant N - J}
     2^{(\alpha - \beta) j} \rho^{a - b} (2^{- j - J} m) \]
  \[ \lesssim \| f \|_{B^{\beta, \varepsilon}_{\infty, \infty} (\rho^b)}
     \sup_{k \geqslant - 1} 2^{(\alpha - \beta) L_k} \rho^{a - b} (2^k), \]
  where we used the fact that $a < b$, $2^{- j} < 2^{- L_k}$ and that the
  weight is decreasing to get
  \[ \rho^{a - b} (2^{- j - J} m) \lesssim \rho^{a - b} (2^{- L_k - J} 2^k)
     \lesssim \rho^{a - b} (2^k) . \]
  Now we set $c_k = - \log_2 \rho (2^k)$ to obtain
  \begin{equation}
    \left\| \UU_{>}^{\varepsilon} f \right\|_{B^{\alpha, \varepsilon}_{\infty,
    \infty} (\rho^a)} \lesssim \| f \|_{B^{\beta, \varepsilon}_{\infty,
    \infty} (\rho^b)} \sup_{k \geqslant - 1} 2^{- (\beta - \alpha) L_k + (b -
    a) c_k} . \label{eq:26}
  \end{equation}

  On the other hand, since $\gamma > \beta$ and $c > b$ we have by the same
  arguments
  \[ \left\| \UU^{\varepsilon}_{\leqslant} f \right\|_{B^{\gamma,
     \varepsilon}_{\infty, \infty} (\rho^c)} \lesssim \sup_{- 1 \leqslant j
     \leqslant N - J} 2^{\gamma j} \sup_{m \in \mathbb{Z}^d} \rho^c (2^{- j -
     J} m) \left| \lambda_{j, m} \left( \UU_{\leqslant}^{\varepsilon} f
     \right) \right| \]
  \[ = \sup_{k \geqslant - 1} \sup_{m \sim 2^k, - 1 \leqslant j \leqslant L_k
     \wedge (N - J)} 2^{(\gamma - \beta) j} \rho^{c - b} (2^{- j - J} m)
     2^{\beta j} \rho^b (2^{- j - J} m) | \lambda_{j, m} (f) | \]
  \begin{equation}
    \lesssim \| f \|_{B^{\beta, \varepsilon}_{\infty, \infty} (\rho^b)}
    \sup_{k \geqslant - 1} 2^{(\gamma - \beta) L_k - (c - b) c_k} .
    \label{eq:27}
  \end{equation}
  We see that if the weight is decreasing at infinity, then $c_k
  \rightarrow \infty$. From {\eqref{eq:26}} we obtain the condition $- (\beta
  - \alpha) L_k + (b - a) c_k = - (\beta - \alpha) L$ hence we shall choose
  $L_k = L + (b - a) c_k / (\beta - \alpha)$. Similarly, {\eqref{eq:27}}
  yields $(\gamma - \beta) L_k - (c - b) c_k = (\gamma - \beta) L$ hence $L_k
  = L + (c - b) c_k / (\gamma - \beta)$. Balancing these two conditions gives
  $(b - a) / (\beta - \alpha) = (c - b) / (\gamma - \beta)$ and completes the
  proof.
  \end{proof}

  \subsection{Duality and commutators}
  \label{s:l2}
  
In this section we define various commutators and establish suitable bounds.  We denote by $C_{\varepsilon}$  the operator introduced in
Lemma 4.4 {\cite{MP17}}, which for smooth functions satisfies
\begin{equation}\label{eq:ce}
 C_{\varepsilon} (f, g, h) = h \circ (f \prec g) - f (h \circ g) .
 \end{equation}
We recall that if $p, p_1, p_2 \in [1, \infty]$ and $\alpha, \beta, \gamma \in
\mathbb{R}$ are such that $\frac{1}{p} = \frac{1}{p_1} + \frac{1}{p_2}$,
$\alpha + \beta + \gamma > 0$ and $\beta + \gamma \neq 0$, then the following
bound holds
\begin{equation}
  \| C_{\varepsilon} (f, g, h) \|_{B^{\beta + \gamma, \varepsilon}_{p, \infty}
  (\rho_1 \rho_2 \rho_3)} \lesssim \| f \|_{B^{\alpha, \varepsilon}_{p_1,
  \infty} (\rho_1)} \| g \|_{B_{\infty, \infty}^{\beta, \varepsilon}
  (\rho_2)_{}} \| h \|_{B_{p_2, \infty}^{\beta, \varepsilon} (\rho_3)_{}} .
  \label{eq:comm}
\end{equation}
As the next step, we show that $g \succ$ is an approximate adjoint of $g
\circ$ in a suitable sense, as first noted in~\cite{gubinelli_semilinear_2018}.
  
  \begin{lemma}
  \label{lem:dual1}Let $\varepsilon \in \mathcal{A}$. Let $\alpha, \beta,
  \gamma \in \mathbb{R}$ be such that $\alpha, \gamma > 0$, $\beta + \gamma <
  0$ and $\alpha + \beta + \gamma > 0$ and let $\rho_1, \rho_2, \rho_3$ be
  weights and let $\rho = \rho_1 \rho_2 \rho_3$. There exists a bounded
  trilinear operator
  \[ D_{\rho, \varepsilon} (f, g, h) : H^{\alpha, \varepsilon} (\rho_1) \times
     \CC^{\beta, \varepsilon} (\rho_2) \times H^{\gamma, \varepsilon} (\rho_3)
     \rightarrow \mathbb{R} \]
  such that
  \[ | D_{\rho, \varepsilon} (f, g, h) | \lesssim \| f \|_{H^{\alpha,
     \varepsilon} (\rho_1)} \| g \|_{\CC^{\beta, \varepsilon} (\rho_2)} \| h
     \|_{H^{\gamma, \varepsilon} (\rho_3)} \]
  where the proportionality constant is independent of $\varepsilon$, and for
  smooth functions we have
  \[ D_{\rho, \varepsilon} (f, g, h) = \langle \rho f, g \circ h
     \rangle_{\varepsilon} - \langle \rho (f \prec g), h \rangle_{\varepsilon}
     . \]
\end{lemma}

\begin{proof}
  We define
  \[ D_{\rho, \varepsilon} (f, g, h) \assign \langle \rho, C_{\varepsilon} (f,
     g, h) \rangle_{\varepsilon} - \langle \rho, (f \prec g) \succ h
     \rangle_{\varepsilon} - \langle \rho, (f \prec g) \prec h
     \rangle_{\varepsilon}, \]
  where $C_{\varepsilon}$ was defined above. Hence the desired formula holds
  for smooth functions. By {\eqref{eq:comm}} and the paraproduct estimates we
  have
  \[ \| C_{\varepsilon} (f, g, h) \|_{B^{\beta + \gamma - \delta,
     \varepsilon}_{1, 1} (\rho)} \lesssim \| C_{\varepsilon} (f, g, h)
     \|_{B^{\beta + \gamma, \varepsilon}_{1, \infty} (\rho)} \lesssim \| f
     \|_{B^{\alpha, \varepsilon}_{2, \infty} (\rho_1)} \| g \|_{B^{\beta,
     \varepsilon}_{\infty, \infty} (\rho_2)} \| h \|_{B^{\gamma,
     \varepsilon}_{2, \infty} (\rho_3)}, \]
  \[ \| (f \prec g) \succ h \|_{B^{\beta - \delta, \varepsilon}_{1, 1} (\rho)}
     \lesssim \| (f \prec g) \succ h \|_{B^{\beta, \varepsilon}_{1, \infty}
     (\rho)} \lesssim \| f \|_{B^{\alpha, \varepsilon}_{2, \infty} (\rho_1)}
     \| g \|_{B^{\beta, \varepsilon}_{\infty, \infty} (\rho_2)} \| h
     \|_{B^{\gamma, \varepsilon}_{2, \infty} (\rho_3)}, \]
  \[ \| (f \prec g) \prec h \|_{B^{\beta + \gamma - \delta, \varepsilon}_{1,
     1} (\rho)} \lesssim \| (f \prec g) \prec h \|_{B^{\beta + \gamma,
     \varepsilon}_{1, \infty} (\rho)} \lesssim \| f \|_{B^{\alpha,
     \varepsilon}_{2, \infty} (\rho_1)} \| g \|_{B^{\beta,
     \varepsilon}_{\infty, \infty} (\rho_2)} \| h \|_{B^{\gamma,
     \varepsilon}_{2, \infty} (\rho_3)}, \]
  and the right hand side is estimated by
  \[ \| f \|_{B^{\alpha, \varepsilon}_{2, \infty} (\rho_1)} \| g \|_{B^{\beta,
     \varepsilon}_{\infty, \infty} (\rho_2)} \| h \|_{B^{\gamma,
     \varepsilon}_{2, \infty} (\rho_3)} \lesssim \| f \|_{B^{\alpha,
     \varepsilon}_{2, 2} (\rho_1)} \| g \|_{B^{\beta, \varepsilon}_{\infty,
     \infty} (\rho_2)} \| h \|_{B^{\gamma, \varepsilon}_{2, 2} (\rho_3)} . \]
  Consequently,
  \[ | D_{\rho, \varepsilon} (f, g, h) | \lesssim \| 1 \|_{B^{- \beta +
     \delta, \varepsilon}_{\infty, \infty}} \| f \|_{B^{\alpha,
     \varepsilon}_{2, 2} (\rho_1)} \| g \|_{B^{\beta, \varepsilon}_{\infty,
     \infty} (\rho_2)} \| h \|_{B^{\gamma, \varepsilon}_{2, 2} (\rho_3)} \]
  which completes the proof.
  \end{proof}

Next, we show several commutator estimates. To this end, $\Delta_{\varepsilon}$ denotes the discrete
Laplacian on $\Lambda_{\varepsilon}$ and we define the corresponding elliptic
and parabolic operators by $\Q_{\varepsilon} \assign m^{2} -
\Delta_{\varepsilon}$ and $\LL_{\varepsilon} \assign \partial_t +
\Q_{\varepsilon}$, where $m^{2} > 0$.

\begin{lemma}
  \label{lem:comm1}Let $\varepsilon \in \mathcal{A}$. Let $\alpha, \beta,
  \gamma \in \mathbb{R}$ such that $\alpha\in (0,1)$, $\beta + \gamma + 2 < 0$ and $\alpha +
  \beta + \gamma + 2 > 0$. Let $\rho_1, \rho_2, \rho_3$ be space weights and
  let $\rho_4, \rho_5, \rho_6$ be space-time weights. Then there exist bounded
  trilinear operators
  \[ \tilde{C}_{\varepsilon} : H^{\alpha, \varepsilon} (\rho_1) \times
     \CC^{\beta, \varepsilon} (\rho_2) \times \CC^{\gamma + \delta,
     \varepsilon} (\rho_3) \rightarrow H^{\beta + \gamma + 2, \varepsilon}
     (\rho_1 \rho_2 \rho_3), \]
  \[ \bar{C}_{\varepsilon} : C_T \CC^{\alpha, \varepsilon} (\rho_4) \times C_T
     \CC^{\beta, \varepsilon} (\rho_5) \times C_T \CC^{\gamma + \delta,
     \varepsilon} (\rho_6) \rightarrow C_T \CC^{\beta + \gamma + 2,
     \varepsilon} (\rho_4 \rho_5 \rho_6) \]
  such that for every $\delta > 0$
  \[ \| \tilde{C}_{\varepsilon} (f, g, h) \|_{H^{\beta + \gamma + 2,
     \varepsilon} (\rho_1 \rho_2 \rho_3)} \lesssim \| f \|_{H^{\alpha,
     \varepsilon} (\rho_1)} \| g \|_{\CC^{\beta, \varepsilon} (\rho_2)} \| h
     \|_{\CC^{\gamma + \delta, \varepsilon} (\rho_3)}, \]
  \[ \| \bar{C}_{\varepsilon} (f, g, h) \|_{C_T \CC^{\beta + \gamma + 2,
     \varepsilon} (\rho_4 \rho_5 \rho_6)} \]
  \[ \lesssim \big( \| f \|_{C_T \CC^{\alpha, \varepsilon} (\rho_4)} + \| f
     \|_{C_T^{\alpha / 2} L^{\infty, \varepsilon} (\rho_4)} \big) \| g
     \|_{C_T \CC^{\beta, \varepsilon} (\rho_5)} \| h \|_{C_T \CC^{\gamma +
     \delta, \varepsilon} (\rho_6)}, \]
  where the proportionality constants are independent of $\varepsilon$, and
  for smooth functions we have
  \begin{equation}
    \tilde{C}_{\varepsilon} (f, g, h) = h \circ \Q_{\varepsilon}^{- 1} (f
    \prec g) - f \left( h \circ \Q_{\varepsilon}^{- 1} g \right), \label{eq:9}
  \end{equation}
  \[ \bar{C}_{\varepsilon} (f, g, h) = h \circ \LL_{\varepsilon}^{- 1} (f
     \prec g) - f \left( h \circ \LL_{\varepsilon}^{- 1} g \right) . \]
\end{lemma}

\begin{proof}
  First, we define
  \[ \tilde{C}_{\varepsilon} (f, g, h) \assign h \circ \left[
     \Q_{\varepsilon}^{- 1} (f \prec g) - f \prec \Q_{\varepsilon}^{- 1} g
     \right] + C_{\varepsilon} \left( f, \Q_{\varepsilon}^{- 1} g, h \right),
  \]
  where $C_{\varepsilon}$ was introduced above. Hence for smooth functions we
  obtain the desired formula {\eqref{eq:9}}. Moreover, by {\eqref{eq:comm}}
  the operator $C_{\varepsilon}$ can be estimated (uniformly in $\varepsilon$)
  for $\delta > 0$ as
  \[ \left\| C_{\varepsilon} \left( f, \Q_{\varepsilon}^{- 1} g, h \right)
     \right\|_{H^{\beta + \gamma + 2, \varepsilon} (\rho_1 \rho_2 \rho_3)}
     \lesssim \left\| C_{\varepsilon} \left( f, \Q_{\varepsilon}^{- 1} g, h
     \right) \right\|_{B_{2, \infty}^{\beta + \gamma + 2 + \delta,
     \varepsilon} (\rho_1 \rho_2 \rho_3)} \]
  \[ \lesssim \| f \|_{B_{2, \infty}^{\alpha, \varepsilon} (\rho_1)} \| g
     \|_{\CC^{\beta, \varepsilon} (\rho_2)} \| h \|_{\CC^{\gamma + \delta,
     \varepsilon} (\rho_3)} \lesssim \| f \|_{H^{\alpha, \varepsilon}
     (\rho_1)} \| g \|_{\CC^{\beta, \varepsilon} (\rho_2)} \| h
     \|_{\CC^{\gamma + \delta, \varepsilon} (\rho_3)} . \]
  For the first term in $\tilde{C}_{\varepsilon}$ we write
  \[ \Q_{\varepsilon}^{- 1} (f \prec g) - f \prec \Q_{\varepsilon}^{- 1} g =
     \Q_{\varepsilon}^{- 1} \left[ f \prec \Q_{\varepsilon}
     \Q_{\varepsilon}^{- 1} g - \Q_{\varepsilon} \left( f \prec
     \Q_{\varepsilon}^{- 1} g \right) \right] \]
  and as a consequence
  \[ \left\| h \circ \left[ \Q_{\varepsilon}^{- 1} (f \prec g) - f \prec
     \Q_{\varepsilon}^{- 1} g \right] \right\|_{H^{\alpha + \beta + \gamma +
     2, \varepsilon} (\rho_1 \rho_2 \rho_3)} \]
  \[ \lesssim \| h \|_{\CC^{\gamma + \delta, \varepsilon} (\rho_3)} \left\| f
     \prec \Q_{\varepsilon} \Q_{\varepsilon}^{- 1} g - \Q_{\varepsilon} \left(
     f \prec \Q_{\varepsilon}^{- 1} g \right) \right\|_{H^{\alpha + \beta -
     \delta, \varepsilon} (\rho_1 \rho_2)} . \]
  Finally, we observe that due to an argument similar to Lemma 4.9 
  {\cite{MP17}} we may control
  \[ \nabla_{\varepsilon} f \prec \nabla_{\varepsilon} g \assign
     \frac12\big(\Delta_{\varepsilon} (f \prec g) - \Delta_{\varepsilon} f \prec g - f
     \prec \Delta_{\varepsilon} g\big), \]
 hence
  \[ \left\| f \prec \Q_{\varepsilon} \Q_{\varepsilon}^{- 1} g -
     \Q_{\varepsilon} \left( f \prec \Q_{\varepsilon}^{- 1} g \right)
     \right\|_{H^{\alpha + \beta - \delta, \varepsilon} (\rho_1 \rho_2)} \]
  \[ \lesssim \left\| f \prec \Q_{\varepsilon} \Q_{\varepsilon}^{- 1} g -
     \Q_{\varepsilon} \left( f \prec \Q_{\varepsilon}^{- 1} g \right)
     \right\|_{B_{2, \infty}^{\alpha + \beta, \varepsilon} (\rho_1 \rho_2)}
     \lesssim \| f \|_{B_{2, \infty}^{\alpha, \varepsilon} (\rho_1)} \| g
     \|_{\CC^{\beta, \varepsilon} (\rho_2)} \]
  \[ \lesssim \| f \|_{H_{}^{\alpha, \varepsilon} (\rho_1)} \| g
     \|_{\CC^{\beta, \varepsilon} (\rho_2)} . \]
  We proceed similarly for the parabolic commutator $\bar{C}_{\varepsilon}$,
  but include additionally a modified paraproduct given by
  \[ f \precprec g \assign \sum_{1 \leqslant i, j \leqslant N - J, i < j - 1}
     \Delta^{\varepsilon}_i Q_i f \Delta^{\varepsilon}_j g, \]
  where
  \[ Q_i f (t) = \int_{\mathbb{R}} 2^{2 i} Q (2^{2 i} (t - s)) f ((s \vee 0)
     \wedge T) \mathd s \]
  for some smooth, nonnegative, compactly supported function $Q : \mathbb{R}
  \rightarrow \mathbb{R}$ that integrates to $1$. Namely, we define
  \[ \bar{C}_{\varepsilon} (f, g, h) \assign h \circ \left[
     \LL_{\varepsilon}^{- 1} (f \precprec g) - f \precprec
     \LL_{\varepsilon}^{- 1} g \right] + h \circ \left[ \LL_{\varepsilon}^{-
     1} (f \prec g - f \precprec g) \right] \]
  \[ + h \circ \left[ f \precprec \LL_{\varepsilon}^{- 1} g - f \prec
     \LL_{\varepsilon}^{- 1} g \right] + C_{\varepsilon} \left( f,
     \LL_{\varepsilon}^{- 1} g, h \right), \]
  and observe that for smooth functions
  \[ \bar{C}_{\varepsilon} (f, g, h) = h \circ \left[ \LL_{\varepsilon}^{- 1}
     (f \prec g) - f \prec \LL_{\varepsilon}^{- 1} g \right] + \left[ h \circ
     \left( f \prec \LL_{\varepsilon}^{- 1} g \right) - f \left( h \circ
     \LL_{\varepsilon}^{- 1} g \right) \right] \]
  \[ = h \circ \LL_{\varepsilon}^{- 1} (f \prec g) - f \left( h \circ
     \LL_{\varepsilon}^{- 1} g \right), \]
  and the desired bound follows from Lemma 4.9 in {\cite{MP17}} and
  {\eqref{eq:comm}}.
  \end{proof}

\subsection{Extension operators}

\label{s:ext}

In order to construct the Euclidean quantum field theory as a
limit of lattice approximations, we  need a suitable extension operator
that allows to extend distributions defined on the lattice
$\Lambda_{\varepsilon}$ to the full space $\mathbb{R}^d$.
To this end, we proceed as in Section 2.4, page 2072 in \cite{MP17}. Namely, let $\psi$ be a smooth and radially symmetric smear function satisfying the properties 1., 2., 3. on page 2072 in \cite{MP17} and let $\psi^{\varepsilon}(\cdot)=\psi(\varepsilon\cdot)$. We define
$$
\mathcal{E}^{\varepsilon}f:=\mathcal{F}_{\mathbb{R}^{d}}^{-1}\big(\psi^{\varepsilon}(\mathcal{F}_{\Lambda_{\varepsilon}}f)_{\rm ext}\big),\qquad f\in\mathcal{S}'(\Lambda_{\varepsilon}),
$$
where  $(\cdot)_{\rm ext}:\mathcal{S}'((\varepsilon^{-1}\mathbb{T})^{d})\to \mathcal{S}'(\mathbb{R}^{d})$ is the periodic extension operator defined by
$$
g_{\rm ext}(\varphi):=g\left(\sum_{k\in (\varepsilon^{-1}\mathbb{Z})^{d}}\varphi(\cdot-k)\right),\qquad \varphi\in\mathcal{S}(\mathbb{R}^{d}).
$$
With the definition of the Dirac comb distribution $f_{\rm dir}\in\mathcal{S}'(\mathbb{R}^{d})$ as in (10) in \cite{MP17}
$$
f_{\rm dir}=\varepsilon^{d}\sum_{k\in\Lambda_{\varepsilon}}f(k)\delta(\cdot - k),\qquad f\in \mathcal{S}'(\Lambda_{\varepsilon}),
$$
it was observed in (14) in \cite{MP17} that
$$
(\mathcal{F}_{\Lambda_{\varepsilon}}f)_{\rm ext}=\mathcal{F}_{\mathbb{R}^{d}}(f_{\rm dir}).
$$
Hence
\[ \mathcal{E}^{\varepsilon} f = \mathcal{F}_{\mathbb{R}^d}^{- 1} (\psi^{\varepsilon}
   (\mathcal{F}_{\Lambda_{\varepsilon}} f)_{\tmop{ext}}) = (\mathcal{F}_{\mathbb{R}^d}^{- 1}
   \psi^{\varepsilon}) \ast \mathcal{F}_{\mathbb{R}^d}^{- 1} \mathcal{F}_{\mathbb{R}^d}
   (f_{\tmop{dir}}) \backassign w^{\varepsilon} \ast f_{\tmop{dir}} =
   w^{\varepsilon} \ast_{\varepsilon} f, \]
where $w^{\varepsilon}(\cdot)=\mathcal{F}_{\mathbb{R}^d}^{- 1}
   \psi^{\varepsilon}(\cdot)=\varepsilon^{-d}\mathcal{F}_{\mathbb{R}^d}^{- 1}
   \psi(\varepsilon^{-1}\cdot)=:\varepsilon^{-d}w(\varepsilon^{-1}\cdot)\in\mathcal{S}(\mathbb{R}^{d})$. 
   With a slight abuse of notation we used the same notation  $\ast_{\varepsilon}$ as for the convolution on the lattice
$\Lambda_{\varepsilon}$ to denote the operation
$$
(w^{\varepsilon} \ast_{\varepsilon} f) (x):=\varepsilon^{d}\sum_{y\in\Lambda_{\varepsilon}}w^{\varepsilon}(x-y)f(y),\qquad x\in \mathbb{R}^{d},
$$
which defines a function on the full space $\mathbb{R}^{d}$.
Note that since $\psi$ is radially symmetric,  $w$ is radially symmetric as well.

%

The following  result is Lemma 2.24 in \cite{MP17}.

\begin{lemma}
  \label{lem:ext}Let $\alpha \in \mathbb{R}$, $p, q \in [1, \infty]$ and let
  $\rho$ be a weight. Then the operators
  \[ \mathcal{E}^{\varepsilon} : B^{\alpha, \varepsilon}_{p, q} (\rho)
     \rightarrow B^{\alpha}_{p, q} (\rho) \]
  are bounded uniformly in $\varepsilon$.
\end{lemma}

\subsection{A Schauder estimate}

In this section we establish a suitable Schauder-type estimate needed in Section~\ref{s:chi-reg}.

\begin{lemma}
  \label{lem:Pt}Let $\rho$ be a weight and let $P^{\varepsilon}_t = e^{t
  (\Delta_{\varepsilon} - m^2)}$ denote the semigroup generated by
  $\Delta_{\varepsilon} - m^2$. Then there exists $c > 0$ uniform in
  $\varepsilon$ such that for all $- 1 \leqslant j \leqslant N - J$
  \[ \| P^{\varepsilon}_t \Delta^{\varepsilon}_j f \|_{L^{1, \varepsilon}
     (\rho)} \lesssim e^{- t (m^2 + c 2^{2 j})} \| \Delta^{\varepsilon}_j f
     \|_{L^{1, \varepsilon} (\rho)}, \]
  where the proportionality constant does not depend on $\varepsilon$ and $t
  \geqslant 0$.
\end{lemma}

\begin{proof}
  Recall that the discrete Laplacian $\Delta_{\varepsilon}$ acts in the
  Fourier space as
  \[ \mathcal{F} (e^{- t (\Delta_{\varepsilon} - m^2)} f) (k) = e^{- t
     l_{\varepsilon} (k)} \hat{f} (k), \]
  where
  \[ l_{\varepsilon} (k) = m^2 + 4 \sum_{j}\sin^2 (\varepsilon \pi k_{j}) /
     \varepsilon^2 . \]
  Consequently, for $- 1 \leqslant j \leqslant N - J$ we have using the fact
  that $\mathcal{F}^{- 1} (g h) =\mathcal{F}_{\mathbb{R}^d}^{- 1} (g)
  \ast_{\varepsilon} \mathcal{F}^{- 1} (h)$ (where $\mathcal{F}^{- 1}$ denotes
  the inverse Fourier transform on the lattice $\Lambda_{\varepsilon}$) we
  obtain

  \[ \Delta^{\varepsilon}_j [e^{t (m^2 - \Delta_{\varepsilon})} f] = [2^{j d}
     V_j (2^j \cdummy)] \ast_{\varepsilon} \Delta^{\varepsilon}_j f, \]
  where
  \[ V_j (x) \assign \int_{\mathbb{R}^d} e^{i 2 \pi x \cdummy \xi} e^{- t
     l_{\varepsilon} (2^j \xi)} \bar{\varphi} (\xi) \mathd \xi, \]
  where $\bar{\varphi}$ is obtained by a rescaling of $\bar{\varphi}_j =
  \sum_{- 1 \leqslant i < \infty ; i \sim j} \varphi_i$. Next, for $M \in
  \mathbb{N}$ we want to show that
  \begin{equation}
    | (1 + | 2 \pi x |^2)^M V_j (x) | \lesssim e^{- t (m^2 + c 2^{2 j})},
    \qquad x \in \mathbb{R}^d . \label{eq:v}
  \end{equation}
  Indeed, with this in hand we may apply Lemma~\ref{lem:young} to deduce the
  claim.
  
  In order to show {\eqref{eq:v}} we compute
  \[ (1 + | 2 \pi x |^2)^M V_j (x) = \int_{\mathbb{R}^d} [(1 - \Delta_{\xi})^M
     e^{i 2 \pi x \cdummy \xi}] e^{- t l_{\varepsilon} (2^j \xi)}
     \bar{\varphi} (\xi) \mathd \xi \]
  \[ = \int_{\mathbb{R}^d} e^{i 2 \pi x \cdummy \xi} (1 - \Delta_{\xi})^M
     [e^{- t l_{\varepsilon} (2^j \xi)} \bar{\varphi} (\xi)] \mathd \xi \]
  where for a multiindex $\alpha \in \mathbb{N}^d$
  \[ \partial_{\xi}^{\alpha} e^{- t l_{\varepsilon} (2^j \xi)} = e^{- t
     l_{\varepsilon} (2^j \xi)} \sum_{0 \leqslant | \beta | \leqslant | \alpha
     |} c_{\alpha, \beta} \partial_{\xi}^{\beta} l_{\varepsilon} (2^j \xi) \]
  therefore using the bounds from Lemma 3.5 in {\cite{MP17}} we obtain
  \[ | \partial_{\xi}^{\alpha} e^{- t l_{\varepsilon} (2^j \xi)} | \lesssim
     e^{- t m^2} e^{- 2 t c (2^j \xi)^2} \sum_{0 \leqslant | \beta | \leqslant
     | \alpha |} \varepsilon^{(| \beta | - 2) \vee 0} (1 + | 2^j \xi |^2)
     \lesssim e^{- t m^2} e^{- t c (2^j \xi)^2} . \]
  Therefore
  \[ | (1 + | 2 \pi x |^2)^M V_j (x) | \lesssim \int_{\mathbb{R}^d} e^{- t c
     (2^j \xi)^2} \bar{\varphi} (\xi) \mathd \xi \lesssim e^{- t m^2} e^{- t c
     2^{2 j}} \]
  and {\eqref{eq:v}} is proven.
\end{proof}

\begin{lemma}
  \label{lem:reg}Let $\alpha \in \mathbb{R}$ and let $\rho$ be a weight. Let
  $v$ solve
  \[ \LL_{\varepsilon} v = f, \qquad v (0) = v_0 . \]
  Then
  \[ \| v \|_{L_T^1 B_{1, 1}^{\alpha, \varepsilon} (\rho)} \lesssim \| v_0
     \|_{B^{\alpha - 2, \varepsilon}_{1, 1} (\rho)} + \| f \|_{L^1_T B_{1,
     1}^{\alpha - 2, \varepsilon} (\rho)}, \]
  where the proportionality constant does not depend on $T$ and $\varepsilon$.
\end{lemma}

\begin{proof}
  Applying the Littlewood--Paley projectors we obtain
  \[ \Delta^{\varepsilon}_j v (t) = P^{\varepsilon}_t \Delta^{\varepsilon}_j
     v_0 + \int_0^t P^{\varepsilon}_{t - s} \Delta^{\varepsilon}_j f (s)
     \mathd s. \]
  Hence according to Lemma~\ref{lem:Pt} there exists $c > 0$ such that for $-
  1 \leqslant j \leqslant N - J$ and uniformly in $T > 0$ and $\varepsilon$
  \[ \| v \|_{L_T^1 B_{1, 1}^{\alpha, \varepsilon} (\rho)} = \int_0^T \sum_{-
     1 \leqslant j \leqslant N - J} 2^{\alpha j} \| \Delta^{\varepsilon}_j v
     (t) \|_{L^{1, \varepsilon} (\rho)} \mathd t \leqslant \int_0^T \sum_{- 1
     \leqslant j \leqslant N - J} 2^{\alpha j} \| P^{\varepsilon}_t
     \Delta^{\varepsilon}_j v_0 \|_{L^{1, \varepsilon} (\rho)} \mathd t \]
  \[ + \int_0^T \sum_{- 1 \leqslant j \leqslant N - J} 2^{\alpha j} \int_0^t
     \| P^{\varepsilon}_{t - s} \Delta^{\varepsilon}_j f (s) \|_{L^{1,
     \varepsilon} (\rho)} \mathd s \mathd t \]
  \[ \leqslant \sum_{- 1 \leqslant j \leqslant N - J} 2^{\alpha j}
     \int_0^{\infty} e^{- t (m^2 + c 2^{2 j})} \mathd t \|
     \Delta^{\varepsilon}_j v_0 \|_{L^{1, \varepsilon} (\rho)} \]
  \[ + \sum_{- 1 \leqslant j \leqslant N - J} 2^{\alpha j} \int_0^T \left[
     \int_0^{\infty} e^{- (t - s) (m^2 + c 2^{2 j})} \mathd t \right] \|
     \Delta^{\varepsilon}_j f (s) \|_{L^{1, \varepsilon} (\rho)} \mathd s \]
  \[ \lesssim \sum_{- 1 \leqslant j \leqslant N - J} 2^{(\alpha - 2) j} \|
     \Delta_j v_0 \|_{L^{1, \varepsilon} (\rho)} + \sum_{- 1 \leqslant j
     \leqslant N - J} 2^{(\alpha - 2) j} \int_0^T \| \Delta^{\varepsilon}_j f
     (s) \|_{L^{1, \varepsilon} (\rho)} \mathd s \]
  \[ = \| v_0 \|_{B^{\alpha - 2, \varepsilon}_{1, 1} (\rho)} + \| f \|_{L^1_T
     B_{1, 1}^{\alpha - 2, \varepsilon} (\rho)} . \]
\end{proof}

\subsection{Regularity of $\chi_{M,\varepsilon}$}
\label{s:chi-reg}

Finally, we proceed with the proof of the proof of Proposition \ref{prop:reg}. 
\medskip

\noindent\textbf{Proof of Proposition \ref{prop:reg} }
  For notational simplicity we fix the parameter $M$ and omit the dependence
  of the various distributions on $M$ throughout the proof. In addition, the $\lambda$-dependent constants  are always bounded uniformly over  $\lambda\in [0,\lambda_{0}]$ for every  $\lambda_{0}>0$.
  
  In view of
  {\eqref{eq:chi1}} we obtain
  \[ \| \rho^{2 + \sigma} \chi_{\varepsilon} \|_{L_T^{\infty} L^{2,
     \varepsilon}} \leqslant \| \rho^2 \phi_{\varepsilon} \|_{L_T^{\infty}
     L^{2, \varepsilon}} + \| \rho^{2 + \sigma} ( 3 \lambda
     X_{\varepsilon}^{\!\resizebox{0.6em}{!}{
\begin{tikzpicture}
\pgfpathmoveto{\pgfqpoint{0cm}{0cm}}
\pgfpathlineto{\pgfqpoint{1.376cm}{0cm}}
\pgfpathlineto{\pgfqpoint{1.376cm}{1.588cm}}
\pgfpathlineto{\pgfqpoint{0cm}{1.588cm}}
\pgfpathclose
\pgfusepath{clip}
\begin{pgfscope}
\begin{pgfscope}
\pgfpathmoveto{\pgfqpoint{0cm}{0cm}}
\pgfpathlineto{\pgfqpoint{1.376cm}{0cm}}
\pgfpathlineto{\pgfqpoint{1.376cm}{1.588cm}}
\pgfpathlineto{\pgfqpoint{0cm}{1.588cm}}
\pgfpathclose
\pgfusepath{clip}
\begin{pgfscope}
\begin{pgfscope}
\definecolor{eps2pgf_color}{gray}{0.976471}\pgfsetstrokecolor{eps2pgf_color}\pgfsetfillcolor{eps2pgf_color}
\pgfpathmoveto{\pgfqpoint{0cm}{0cm}}
\pgfpathlineto{\pgfqpoint{1.376cm}{0cm}}
\pgfpathlineto{\pgfqpoint{1.376cm}{1.588cm}}
\pgfpathlineto{\pgfqpoint{0cm}{1.588cm}}
\pgfpathclose
\pgfusepath{fill}
\end{pgfscope}
\begin{pgfscope}
\pgfsetdash{}{0cm}
\pgfsetlinewidth{0.818mm}
\pgfsetroundcap
\pgfsetroundjoin
\pgfsetmiterlimit{7.0}
\definecolor{eps2pgf_color}{gray}{0}\pgfsetstrokecolor{eps2pgf_color}\pgfsetfillcolor{eps2pgf_color}
\pgfpathmoveto{\pgfqpoint{0.117cm}{1.476cm}}
\pgfpathlineto{\pgfqpoint{0.682cm}{0.726cm}}
\pgfpathlineto{\pgfqpoint{1.246cm}{1.476cm}}
\pgfusepath{stroke}
\end{pgfscope}
\definecolor{eps2pgf_color}{gray}{0}\pgfsetstrokecolor{eps2pgf_color}\pgfsetfillcolor{eps2pgf_color}
\pgfpathmoveto{\pgfqpoint{0.273cm}{1.451cm}}
\pgfpathcurveto{\pgfqpoint{0.273cm}{1.487cm}}{\pgfqpoint{0.259cm}{1.522cm}}{\pgfqpoint{0.233cm}{1.547cm}}
\pgfpathcurveto{\pgfqpoint{0.207cm}{1.573cm}}{\pgfqpoint{0.173cm}{1.588cm}}{\pgfqpoint{0.137cm}{1.588cm}}
\pgfpathcurveto{\pgfqpoint{0.1cm}{1.588cm}}{\pgfqpoint{0.066cm}{1.573cm}}{\pgfqpoint{0.04cm}{1.547cm}}
\pgfpathcurveto{\pgfqpoint{0.014cm}{1.522cm}}{\pgfqpoint{0cm}{1.487cm}}{\pgfqpoint{0cm}{1.451cm}}
\pgfpathcurveto{\pgfqpoint{0cm}{1.414cm}}{\pgfqpoint{0.014cm}{1.379cm}}{\pgfqpoint{0.04cm}{1.354cm}}
\pgfpathcurveto{\pgfqpoint{0.066cm}{1.328cm}}{\pgfqpoint{0.1cm}{1.314cm}}{\pgfqpoint{0.137cm}{1.314cm}}
\pgfpathcurveto{\pgfqpoint{0.173cm}{1.314cm}}{\pgfqpoint{0.207cm}{1.328cm}}{\pgfqpoint{0.233cm}{1.354cm}}
\pgfpathcurveto{\pgfqpoint{0.259cm}{1.379cm}}{\pgfqpoint{0.273cm}{1.414cm}}{\pgfqpoint{0.273cm}{1.451cm}}
\pgfusepath{fill}
\pgfpathmoveto{\pgfqpoint{1.345cm}{1.426cm}}
\pgfpathcurveto{\pgfqpoint{1.345cm}{1.463cm}}{\pgfqpoint{1.331cm}{1.497cm}}{\pgfqpoint{1.305cm}{1.523cm}}
\pgfpathcurveto{\pgfqpoint{1.28cm}{1.549cm}}{\pgfqpoint{1.245cm}{1.563cm}}{\pgfqpoint{1.209cm}{1.563cm}}
\pgfpathcurveto{\pgfqpoint{1.172cm}{1.563cm}}{\pgfqpoint{1.138cm}{1.549cm}}{\pgfqpoint{1.112cm}{1.523cm}}
\pgfpathcurveto{\pgfqpoint{1.087cm}{1.497cm}}{\pgfqpoint{1.072cm}{1.463cm}}{\pgfqpoint{1.072cm}{1.426cm}}
\pgfpathcurveto{\pgfqpoint{1.072cm}{1.39cm}}{\pgfqpoint{1.087cm}{1.355cm}}{\pgfqpoint{1.112cm}{1.329cm}}
\pgfpathcurveto{\pgfqpoint{1.138cm}{1.304cm}}{\pgfqpoint{1.172cm}{1.289cm}}{\pgfqpoint{1.209cm}{1.289cm}}
\pgfpathcurveto{\pgfqpoint{1.245cm}{1.289cm}}{\pgfqpoint{1.28cm}{1.304cm}}{\pgfqpoint{1.305cm}{1.329cm}}
\pgfpathcurveto{\pgfqpoint{1.331cm}{1.355cm}}{\pgfqpoint{1.345cm}{1.39cm}}{\pgfqpoint{1.345cm}{1.426cm}}
\pgfusepath{fill}
\begin{pgfscope}
\pgfsetdash{}{0cm}
\pgfsetlinewidth{0.818mm}
\pgfsetroundcap
\pgfsetmiterlimit{4.0}
\pgfpathmoveto{\pgfqpoint{0.682cm}{0.726cm}}
\pgfpathlineto{\pgfqpoint{0.682cm}{0.097cm}}
\pgfusepath{stroke}
\end{pgfscope}
\end{pgfscope}
\end{pgfscope}
\end{pgfscope}
\end{tikzpicture}}} \succ \phi_{\varepsilon} )
     \|_{L_T^{\infty} L^{2, \varepsilon}} \leqslant C_{\lambda}\| \rho^2
     \phi_{\varepsilon} \|_{L_T^{\infty} L^{2, \varepsilon}} Q_{\rho}
     (\mathbb{X}_{\varepsilon}), \]
  where, by Theorem~\ref{th:energy-estimate},
  \begin{equation*}
  \begin{aligned}
  \| \rho^2 \phi_{\varepsilon} (t) \|_{L^{2, \varepsilon}}^2 & \leqslant 
     C_{t,\lambda} Q_{\rho} (\mathbb{X}_{\varepsilon})
     + \| \rho^2 \phi_{ \varepsilon} (0)\|_{L^{2, \varepsilon}}^2.
  \end{aligned}
  \end{equation*}
  Thus
  \begin{equation}
    \| \rho^{2 + \sigma} \chi_{\varepsilon} \|_{L_T^{\infty} L^{2,
    \varepsilon}} \leqslant C_{T,\lambda}Q_{\rho} (\mathbb{X}_{\varepsilon}) (1  +\|
    \rho^2 \phi_{\varepsilon} (0)\|_{L^{2, \varepsilon}}) . \label{eq:chi}
  \end{equation}
  Next, we intend to apply Lemma~\ref{lem:reg} to {\eqref{eq:chi11}} in the
  form
  \[ \| \rho^4 \chi_{\varepsilon} \|_{L^1_T B_{1, 1}^{1 + 3 \kappa,
     \varepsilon}} \lesssim \| \rho^4 \chi_{\varepsilon} (0) \|_{B_{1, 1}^{- 1
     + 3 \kappa, \varepsilon}} + \left\| \rho^4 \LL_{\varepsilon}
     \chi_{\varepsilon} \right\|_{L_T^1 B_{1, 1}^{- 1 + 3 \kappa,
     \varepsilon}}. \]
In view of
  the second term on the  right hand side of {\eqref{eq:chi11}} we shall
  therefore estimate $U_{\varepsilon}$ in $B_{1, 1}^{- 1 + 3 \kappa,
  \varepsilon} (\rho^{4 - \sigma})$ as the weight $\rho^{\sigma}$ will be lost
  to control $X^{\!\resizebox{0.6em}{!}{
\begin{tikzpicture}
\pgfpathmoveto{\pgfqpoint{0cm}{0cm}}
\pgfpathlineto{\pgfqpoint{1.376cm}{0cm}}
\pgfpathlineto{\pgfqpoint{1.376cm}{1.588cm}}
\pgfpathlineto{\pgfqpoint{0cm}{1.588cm}}
\pgfpathclose
\pgfusepath{clip}
\begin{pgfscope}
\begin{pgfscope}
\pgfpathmoveto{\pgfqpoint{0cm}{0cm}}
\pgfpathlineto{\pgfqpoint{1.376cm}{0cm}}
\pgfpathlineto{\pgfqpoint{1.376cm}{1.588cm}}
\pgfpathlineto{\pgfqpoint{0cm}{1.588cm}}
\pgfpathclose
\pgfusepath{clip}
\begin{pgfscope}
\begin{pgfscope}
\definecolor{eps2pgf_color}{gray}{0.976471}\pgfsetstrokecolor{eps2pgf_color}\pgfsetfillcolor{eps2pgf_color}
\pgfpathmoveto{\pgfqpoint{0cm}{0cm}}
\pgfpathlineto{\pgfqpoint{1.376cm}{0cm}}
\pgfpathlineto{\pgfqpoint{1.376cm}{1.588cm}}
\pgfpathlineto{\pgfqpoint{0cm}{1.588cm}}
\pgfpathclose
\pgfusepath{fill}
\end{pgfscope}
\begin{pgfscope}
\pgfsetdash{}{0cm}
\pgfsetlinewidth{0.818mm}
\pgfsetroundcap
\pgfsetroundjoin
\pgfsetmiterlimit{7.0}
\definecolor{eps2pgf_color}{gray}{0}\pgfsetstrokecolor{eps2pgf_color}\pgfsetfillcolor{eps2pgf_color}
\pgfpathmoveto{\pgfqpoint{0.117cm}{1.476cm}}
\pgfpathlineto{\pgfqpoint{0.682cm}{0.726cm}}
\pgfpathlineto{\pgfqpoint{1.246cm}{1.476cm}}
\pgfusepath{stroke}
\end{pgfscope}
\definecolor{eps2pgf_color}{gray}{0}\pgfsetstrokecolor{eps2pgf_color}\pgfsetfillcolor{eps2pgf_color}
\pgfpathmoveto{\pgfqpoint{0.273cm}{1.451cm}}
\pgfpathcurveto{\pgfqpoint{0.273cm}{1.487cm}}{\pgfqpoint{0.259cm}{1.522cm}}{\pgfqpoint{0.233cm}{1.547cm}}
\pgfpathcurveto{\pgfqpoint{0.207cm}{1.573cm}}{\pgfqpoint{0.173cm}{1.588cm}}{\pgfqpoint{0.137cm}{1.588cm}}
\pgfpathcurveto{\pgfqpoint{0.1cm}{1.588cm}}{\pgfqpoint{0.066cm}{1.573cm}}{\pgfqpoint{0.04cm}{1.547cm}}
\pgfpathcurveto{\pgfqpoint{0.014cm}{1.522cm}}{\pgfqpoint{0cm}{1.487cm}}{\pgfqpoint{0cm}{1.451cm}}
\pgfpathcurveto{\pgfqpoint{0cm}{1.414cm}}{\pgfqpoint{0.014cm}{1.379cm}}{\pgfqpoint{0.04cm}{1.354cm}}
\pgfpathcurveto{\pgfqpoint{0.066cm}{1.328cm}}{\pgfqpoint{0.1cm}{1.314cm}}{\pgfqpoint{0.137cm}{1.314cm}}
\pgfpathcurveto{\pgfqpoint{0.173cm}{1.314cm}}{\pgfqpoint{0.207cm}{1.328cm}}{\pgfqpoint{0.233cm}{1.354cm}}
\pgfpathcurveto{\pgfqpoint{0.259cm}{1.379cm}}{\pgfqpoint{0.273cm}{1.414cm}}{\pgfqpoint{0.273cm}{1.451cm}}
\pgfusepath{fill}
\pgfpathmoveto{\pgfqpoint{1.345cm}{1.426cm}}
\pgfpathcurveto{\pgfqpoint{1.345cm}{1.463cm}}{\pgfqpoint{1.331cm}{1.497cm}}{\pgfqpoint{1.305cm}{1.523cm}}
\pgfpathcurveto{\pgfqpoint{1.28cm}{1.549cm}}{\pgfqpoint{1.245cm}{1.563cm}}{\pgfqpoint{1.209cm}{1.563cm}}
\pgfpathcurveto{\pgfqpoint{1.172cm}{1.563cm}}{\pgfqpoint{1.138cm}{1.549cm}}{\pgfqpoint{1.112cm}{1.523cm}}
\pgfpathcurveto{\pgfqpoint{1.087cm}{1.497cm}}{\pgfqpoint{1.072cm}{1.463cm}}{\pgfqpoint{1.072cm}{1.426cm}}
\pgfpathcurveto{\pgfqpoint{1.072cm}{1.39cm}}{\pgfqpoint{1.087cm}{1.355cm}}{\pgfqpoint{1.112cm}{1.329cm}}
\pgfpathcurveto{\pgfqpoint{1.138cm}{1.304cm}}{\pgfqpoint{1.172cm}{1.289cm}}{\pgfqpoint{1.209cm}{1.289cm}}
\pgfpathcurveto{\pgfqpoint{1.245cm}{1.289cm}}{\pgfqpoint{1.28cm}{1.304cm}}{\pgfqpoint{1.305cm}{1.329cm}}
\pgfpathcurveto{\pgfqpoint{1.331cm}{1.355cm}}{\pgfqpoint{1.345cm}{1.39cm}}{\pgfqpoint{1.345cm}{1.426cm}}
\pgfusepath{fill}
\begin{pgfscope}
\pgfsetdash{}{0cm}
\pgfsetlinewidth{0.818mm}
\pgfsetroundcap
\pgfsetmiterlimit{4.0}
\pgfpathmoveto{\pgfqpoint{0.682cm}{0.726cm}}
\pgfpathlineto{\pgfqpoint{0.682cm}{0.097cm}}
\pgfusepath{stroke}
\end{pgfscope}
\end{pgfscope}
\end{pgfscope}
\end{pgfscope}
\end{tikzpicture}}}$. Let us first show how to bound the terms that
  contain higher powers of $\phi$, all the other terms being straightforward.
  By paraproduct estimates Lemma~\ref{lem:mult} and Lemma~\ref{lem:15}, we
  obtain
  \[ \| \rho^{4 - \sigma} \lambda  X_{\varepsilon} \phi_{\varepsilon}^2 \|_{B_{1,
     1}^{- 1 + 3 \kappa, \varepsilon}} \lesssim \lambda  \| \rho^{\sigma}
     X_{\varepsilon} \|_{\CC^{- 1 / 2 - \kappa, \varepsilon}} \| \rho^{4 - 2
     \sigma} \phi_{\varepsilon}^2 \|^{}_{B_{1, 1}^{1 / 2 + 2 \kappa,
     \varepsilon}} \]
  \[ \lesssim \lambda  \| \rho^{\sigma} X_{\varepsilon} \|_{\CC^{- 1 / 2 - \kappa,
     \varepsilon}} \| \rho^{1 + \iota} \phi_{\varepsilon} \|_{L^{2,
     \varepsilon}} \| \rho^2 \phi_{\varepsilon} \|_{H^{1 / 2 + 3 \kappa,
     \varepsilon}} \leqslant \lambda  Q_{\rho} (\mathbb{X}_{\varepsilon}) \| \rho
     \phi_{\varepsilon} \|_{L^{4, \varepsilon}} \| \rho^2 \phi_{\varepsilon}
     \|_{H^{1 - 2 \kappa, \varepsilon}} \]
  while   
  \[ \| \rho^{4 - \sigma} 3\lambda  Y_{\varepsilon} \phi_{\varepsilon}^2 \|_{B_{1,
     1}^{- 1 + 3 \kappa, \varepsilon}} \lesssim \lambda  \| \rho^{\sigma}
     Y_{\varepsilon} \|_{\CC^{1 / 2 - \kappa, \varepsilon}} \| \rho^{4 - 2
     \sigma} \phi_{\varepsilon}^2 \|_{B^{\kappa, \varepsilon}_{1, 1}} \]
  \[ \lesssim \lambda \| \rho^{\sigma} Y_{\varepsilon} \|_{\CC^{1 / 2 - \kappa,
     \varepsilon}} \| \rho^{1 + \iota} \phi_{\varepsilon} \|_{L^{2,
     \varepsilon}} \| \rho^2 \phi_{\varepsilon} \|_{H^{2 \kappa, \varepsilon}}
     \leqslant \lambda^2 Q_{\rho} (\mathbb{X}_{\varepsilon}) \| \rho \phi_{\varepsilon}
     \|_{L^{4, \varepsilon}} \| \rho^2 \phi_{\varepsilon} \|_{H^{1 - 2 \kappa,
     \varepsilon}}, \]
  and by interpolation for $\theta = \frac{1 - 4 \kappa}{1 - 2 \kappa}$
  \[ \| \rho^{4 - \sigma} \lambda \phi_{\varepsilon}^3 \|_{B_{1, 1}^{- 1 + 3 \kappa,
     \varepsilon}} \lesssim \lambda \| \rho^{4 - \sigma} \phi_{\varepsilon}^3
     \|_{B_{1, 1}^{\kappa, \varepsilon}} \lesssim \lambda  \| \rho \phi_{\varepsilon}
     \|_{L^{4, \varepsilon}}^2 \| \rho^{2 - \sigma} \phi_{\varepsilon}
     \|_{H^{2 \kappa, \varepsilon}} \]
  \[ \lesssim \lambda  \| \rho \phi_{\varepsilon} \|_{L^{4, \varepsilon}}^2 \| \rho^{1
     + \iota} \phi_{\varepsilon} \|_{L^{2, \varepsilon}}^{\theta} \| \rho^2
     \phi_{\varepsilon} \|_{H^{1 - 2 \kappa, \varepsilon}}^{1 - \theta}
     \lesssim\lambda  \| \rho \phi_{\varepsilon} \|_{L^{4, \varepsilon}}^{2 + \theta}
     \| \rho^2 \phi_{\varepsilon} \|_{H^{1 - 2 \kappa, \varepsilon}}^{1 -
     \theta} . \]
  Consequently, we use the embeddings $B_{2, 2}^{\alpha + \kappa, \varepsilon}
  (\rho^{2 + \beta}) \subset B_{1, 1}^{\alpha, \varepsilon} (\rho^{4 -
  \sigma})$ and $B_{\infty, \infty}^{\alpha + \kappa, \varepsilon}
  (\rho^{\beta}) \subset B_{1, 1}^{\alpha, \varepsilon} (\rho^{4 - \sigma})$
  for $\alpha \in \mathbb{R}$ (provided the weight possesses enough
  integrability and $\beta, \sigma > 0$ are sufficiently small). We deduce
  \[ \begin{aligned}
       \| \rho^{4 - \sigma} U_{\varepsilon} \|_{B_{1, 1}^{- 1 + 3 \kappa,
       \varepsilon}} & \lesssim  \lambda^2  \| \rho^{\sigma}
       \tilde{X}_{\varepsilon}^{\!\resizebox{!}{.8em}{
\begin{tikzpicture}
\pgfpathmoveto{\pgfqpoint{0cm}{-0.035cm}}
\pgfpathlineto{\pgfqpoint{1.976cm}{-0.035cm}}
\pgfpathlineto{\pgfqpoint{1.976cm}{1.94cm}}
\pgfpathlineto{\pgfqpoint{0cm}{1.94cm}}
\pgfpathclose
\pgfusepath{clip}
\begin{pgfscope}
\begin{pgfscope}
\pgfpathmoveto{\pgfqpoint{0cm}{-0.035cm}}
\pgfpathlineto{\pgfqpoint{1.976cm}{-0.035cm}}
\pgfpathlineto{\pgfqpoint{1.976cm}{1.94cm}}
\pgfpathlineto{\pgfqpoint{0cm}{1.94cm}}
\pgfpathclose
\pgfusepath{clip}
\begin{pgfscope}
\begin{pgfscope}
\pgfsetdash{}{0cm}
\pgfsetlinewidth{0.818mm}
\pgfsetroundcap
\pgfsetroundjoin
\pgfsetmiterlimit{7.0}
\definecolor{eps2pgf_color}{gray}{0}\pgfsetstrokecolor{eps2pgf_color}\pgfsetfillcolor{eps2pgf_color}
\pgfpathmoveto{\pgfqpoint{0.117cm}{1.815cm}}
\pgfpathlineto{\pgfqpoint{0.682cm}{1.065cm}}
\pgfpathlineto{\pgfqpoint{1.246cm}{1.815cm}}
\pgfusepath{stroke}
\end{pgfscope}
\definecolor{eps2pgf_color}{gray}{0}\pgfsetstrokecolor{eps2pgf_color}\pgfsetfillcolor{eps2pgf_color}
\pgfpathmoveto{\pgfqpoint{0.273cm}{1.789cm}}
\pgfpathcurveto{\pgfqpoint{0.273cm}{1.825cm}}{\pgfqpoint{0.259cm}{1.86cm}}{\pgfqpoint{0.233cm}{1.886cm}}
\pgfpathcurveto{\pgfqpoint{0.207cm}{1.912cm}}{\pgfqpoint{0.173cm}{1.926cm}}{\pgfqpoint{0.137cm}{1.926cm}}
\pgfpathcurveto{\pgfqpoint{0.1cm}{1.926cm}}{\pgfqpoint{0.066cm}{1.912cm}}{\pgfqpoint{0.04cm}{1.886cm}}
\pgfpathcurveto{\pgfqpoint{0.014cm}{1.86cm}}{\pgfqpoint{0cm}{1.825cm}}{\pgfqpoint{0cm}{1.789cm}}
\pgfpathcurveto{\pgfqpoint{0cm}{1.753cm}}{\pgfqpoint{0.014cm}{1.718cm}}{\pgfqpoint{0.04cm}{1.692cm}}
\pgfpathcurveto{\pgfqpoint{0.066cm}{1.667cm}}{\pgfqpoint{0.1cm}{1.652cm}}{\pgfqpoint{0.137cm}{1.652cm}}
\pgfpathcurveto{\pgfqpoint{0.173cm}{1.652cm}}{\pgfqpoint{0.207cm}{1.667cm}}{\pgfqpoint{0.233cm}{1.692cm}}
\pgfpathcurveto{\pgfqpoint{0.259cm}{1.718cm}}{\pgfqpoint{0.273cm}{1.753cm}}{\pgfqpoint{0.273cm}{1.789cm}}
\pgfusepath{fill}
\pgfpathmoveto{\pgfqpoint{1.345cm}{1.765cm}}
\pgfpathcurveto{\pgfqpoint{1.345cm}{1.801cm}}{\pgfqpoint{1.331cm}{1.836cm}}{\pgfqpoint{1.305cm}{1.862cm}}
\pgfpathcurveto{\pgfqpoint{1.28cm}{1.887cm}}{\pgfqpoint{1.245cm}{1.902cm}}{\pgfqpoint{1.209cm}{1.902cm}}
\pgfpathcurveto{\pgfqpoint{1.172cm}{1.902cm}}{\pgfqpoint{1.138cm}{1.887cm}}{\pgfqpoint{1.112cm}{1.862cm}}
\pgfpathcurveto{\pgfqpoint{1.087cm}{1.836cm}}{\pgfqpoint{1.072cm}{1.801cm}}{\pgfqpoint{1.072cm}{1.765cm}}
\pgfpathcurveto{\pgfqpoint{1.072cm}{1.728cm}}{\pgfqpoint{1.087cm}{1.694cm}}{\pgfqpoint{1.112cm}{1.668cm}}
\pgfpathcurveto{\pgfqpoint{1.138cm}{1.642cm}}{\pgfqpoint{1.172cm}{1.628cm}}{\pgfqpoint{1.209cm}{1.628cm}}
\pgfpathcurveto{\pgfqpoint{1.245cm}{1.628cm}}{\pgfqpoint{1.28cm}{1.642cm}}{\pgfqpoint{1.305cm}{1.668cm}}
\pgfpathcurveto{\pgfqpoint{1.331cm}{1.694cm}}{\pgfqpoint{1.345cm}{1.728cm}}{\pgfqpoint{1.345cm}{1.765cm}}
\pgfusepath{fill}
\begin{pgfscope}
\pgfsetdash{}{0cm}
\pgfsetlinewidth{0.818mm}
\pgfsetroundcap
\pgfsetroundjoin
\pgfsetmiterlimit{7.0}
\pgfpathmoveto{\pgfqpoint{0.682cm}{1.065cm}}
\pgfpathlineto{\pgfqpoint{1.246cm}{0.315cm}}
\pgfpathlineto{\pgfqpoint{1.811cm}{1.065cm}}
\pgfusepath{stroke}
\end{pgfscope}
\pgfpathmoveto{\pgfqpoint{1.948cm}{1.065cm}}
\pgfpathcurveto{\pgfqpoint{1.948cm}{1.101cm}}{\pgfqpoint{1.933cm}{1.136cm}}{\pgfqpoint{1.907cm}{1.162cm}}
\pgfpathcurveto{\pgfqpoint{1.882cm}{1.187cm}}{\pgfqpoint{1.847cm}{1.202cm}}{\pgfqpoint{1.811cm}{1.202cm}}
\pgfpathcurveto{\pgfqpoint{1.775cm}{1.202cm}}{\pgfqpoint{1.74cm}{1.187cm}}{\pgfqpoint{1.714cm}{1.162cm}}
\pgfpathcurveto{\pgfqpoint{1.689cm}{1.136cm}}{\pgfqpoint{1.674cm}{1.101cm}}{\pgfqpoint{1.674cm}{1.065cm}}
\pgfpathcurveto{\pgfqpoint{1.674cm}{1.029cm}}{\pgfqpoint{1.689cm}{0.994cm}}{\pgfqpoint{1.714cm}{0.968cm}}
\pgfpathcurveto{\pgfqpoint{1.74cm}{0.942cm}}{\pgfqpoint{1.775cm}{0.928cm}}{\pgfqpoint{1.811cm}{0.928cm}}
\pgfpathcurveto{\pgfqpoint{1.847cm}{0.928cm}}{\pgfqpoint{1.882cm}{0.942cm}}{\pgfqpoint{1.907cm}{0.968cm}}
\pgfpathcurveto{\pgfqpoint{1.933cm}{0.994cm}}{\pgfqpoint{1.948cm}{1.029cm}}{\pgfqpoint{1.948cm}{1.065cm}}
\pgfusepath{fill}
\begin{pgfscope}
\pgfsetdash{}{0cm}
\pgfsetlinewidth{0.818mm}
\pgfsetmiterlimit{7.0}
\pgfpathmoveto{\pgfqpoint{1.246cm}{0.315cm}}
\pgfpathlineto{\pgfqpoint{1.244cm}{1.061cm}}
\pgfusepath{stroke}
\end{pgfscope}
\pgfpathmoveto{\pgfqpoint{1.38cm}{1.065cm}}
\pgfpathcurveto{\pgfqpoint{1.38cm}{1.101cm}}{\pgfqpoint{1.366cm}{1.136cm}}{\pgfqpoint{1.34cm}{1.162cm}}
\pgfpathcurveto{\pgfqpoint{1.315cm}{1.187cm}}{\pgfqpoint{1.28cm}{1.202cm}}{\pgfqpoint{1.244cm}{1.202cm}}
\pgfpathcurveto{\pgfqpoint{1.207cm}{1.202cm}}{\pgfqpoint{1.173cm}{1.187cm}}{\pgfqpoint{1.147cm}{1.162cm}}
\pgfpathcurveto{\pgfqpoint{1.121cm}{1.136cm}}{\pgfqpoint{1.107cm}{1.101cm}}{\pgfqpoint{1.107cm}{1.065cm}}
\pgfpathcurveto{\pgfqpoint{1.107cm}{1.029cm}}{\pgfqpoint{1.121cm}{0.994cm}}{\pgfqpoint{1.147cm}{0.968cm}}
\pgfpathcurveto{\pgfqpoint{1.173cm}{0.942cm}}{\pgfqpoint{1.207cm}{0.928cm}}{\pgfqpoint{1.244cm}{0.928cm}}
\pgfpathcurveto{\pgfqpoint{1.28cm}{0.928cm}}{\pgfqpoint{1.315cm}{0.942cm}}{\pgfqpoint{1.34cm}{0.968cm}}
\pgfpathcurveto{\pgfqpoint{1.366cm}{0.994cm}}{\pgfqpoint{1.38cm}{1.029cm}}{\pgfqpoint{1.38cm}{1.065cm}}
\pgfusepath{fill}
\begin{pgfscope}
\pgfsetdash{}{0cm}
\pgfsetlinewidth{0.818mm}
\pgfsetmiterlimit{4.0}
\pgfpathmoveto{\pgfqpoint{1.383cm}{0.178cm}}
\pgfpathcurveto{\pgfqpoint{1.383cm}{0.214cm}}{\pgfqpoint{1.369cm}{0.249cm}}{\pgfqpoint{1.343cm}{0.275cm}}
\pgfpathcurveto{\pgfqpoint{1.317cm}{0.3cm}}{\pgfqpoint{1.283cm}{0.315cm}}{\pgfqpoint{1.246cm}{0.315cm}}
\pgfpathcurveto{\pgfqpoint{1.21cm}{0.315cm}}{\pgfqpoint{1.175cm}{0.3cm}}{\pgfqpoint{1.15cm}{0.275cm}}
\pgfpathcurveto{\pgfqpoint{1.124cm}{0.249cm}}{\pgfqpoint{1.11cm}{0.214cm}}{\pgfqpoint{1.11cm}{0.178cm}}
\pgfpathcurveto{\pgfqpoint{1.11cm}{0.141cm}}{\pgfqpoint{1.124cm}{0.107cm}}{\pgfqpoint{1.15cm}{0.081cm}}
\pgfpathcurveto{\pgfqpoint{1.175cm}{0.055cm}}{\pgfqpoint{1.21cm}{0.041cm}}{\pgfqpoint{1.246cm}{0.041cm}}
\pgfpathcurveto{\pgfqpoint{1.283cm}{0.041cm}}{\pgfqpoint{1.317cm}{0.055cm}}{\pgfqpoint{1.343cm}{0.081cm}}
\pgfpathcurveto{\pgfqpoint{1.369cm}{0.107cm}}{\pgfqpoint{1.383cm}{0.141cm}}{\pgfqpoint{1.383cm}{0.178cm}}
\pgfusepath{stroke}
\end{pgfscope}
\end{pgfscope}
\end{pgfscope}
\end{pgfscope}
\end{tikzpicture}}} \|_{\CC^{- \kappa, \varepsilon}} \| \rho^2
       \phi_{\varepsilon} \|_{H^{1 - 2 \kappa, \varepsilon}} + \lambda^2 | \log t | \|
       \rho^2 \phi_{\varepsilon} \|_{H^{1 - 2 \kappa, \varepsilon}}\\
       &  \quad + \lambda^2 \| \rho^{\sigma} \llbracket X_{\varepsilon}^2 \rrbracket
       \|_{\CC^{- 1 - \kappa, \varepsilon}} \| \rho^{\sigma}
       X_{\varepsilon}^{\!\resizebox{0.6em}{!}{
\begin{tikzpicture}
\pgfpathmoveto{\pgfqpoint{0cm}{0cm}}
\pgfpathlineto{\pgfqpoint{1.376cm}{0cm}}
\pgfpathlineto{\pgfqpoint{1.376cm}{1.588cm}}
\pgfpathlineto{\pgfqpoint{0cm}{1.588cm}}
\pgfpathclose
\pgfusepath{clip}
\begin{pgfscope}
\begin{pgfscope}
\pgfpathmoveto{\pgfqpoint{0cm}{0cm}}
\pgfpathlineto{\pgfqpoint{1.376cm}{0cm}}
\pgfpathlineto{\pgfqpoint{1.376cm}{1.588cm}}
\pgfpathlineto{\pgfqpoint{0cm}{1.588cm}}
\pgfpathclose
\pgfusepath{clip}
\begin{pgfscope}
\begin{pgfscope}
\definecolor{eps2pgf_color}{gray}{0.976471}\pgfsetstrokecolor{eps2pgf_color}\pgfsetfillcolor{eps2pgf_color}
\pgfpathmoveto{\pgfqpoint{0cm}{0cm}}
\pgfpathlineto{\pgfqpoint{1.376cm}{0cm}}
\pgfpathlineto{\pgfqpoint{1.376cm}{1.588cm}}
\pgfpathlineto{\pgfqpoint{0cm}{1.588cm}}
\pgfpathclose
\pgfusepath{fill}
\end{pgfscope}
\begin{pgfscope}
\pgfsetdash{}{0cm}
\pgfsetlinewidth{0.818mm}
\pgfsetroundcap
\pgfsetroundjoin
\pgfsetmiterlimit{7.0}
\definecolor{eps2pgf_color}{gray}{0}\pgfsetstrokecolor{eps2pgf_color}\pgfsetfillcolor{eps2pgf_color}
\pgfpathmoveto{\pgfqpoint{0.117cm}{1.476cm}}
\pgfpathlineto{\pgfqpoint{0.682cm}{0.726cm}}
\pgfpathlineto{\pgfqpoint{1.246cm}{1.476cm}}
\pgfusepath{stroke}
\end{pgfscope}
\definecolor{eps2pgf_color}{gray}{0}\pgfsetstrokecolor{eps2pgf_color}\pgfsetfillcolor{eps2pgf_color}
\pgfpathmoveto{\pgfqpoint{0.273cm}{1.451cm}}
\pgfpathcurveto{\pgfqpoint{0.273cm}{1.487cm}}{\pgfqpoint{0.259cm}{1.522cm}}{\pgfqpoint{0.233cm}{1.547cm}}
\pgfpathcurveto{\pgfqpoint{0.207cm}{1.573cm}}{\pgfqpoint{0.173cm}{1.588cm}}{\pgfqpoint{0.137cm}{1.588cm}}
\pgfpathcurveto{\pgfqpoint{0.1cm}{1.588cm}}{\pgfqpoint{0.066cm}{1.573cm}}{\pgfqpoint{0.04cm}{1.547cm}}
\pgfpathcurveto{\pgfqpoint{0.014cm}{1.522cm}}{\pgfqpoint{0cm}{1.487cm}}{\pgfqpoint{0cm}{1.451cm}}
\pgfpathcurveto{\pgfqpoint{0cm}{1.414cm}}{\pgfqpoint{0.014cm}{1.379cm}}{\pgfqpoint{0.04cm}{1.354cm}}
\pgfpathcurveto{\pgfqpoint{0.066cm}{1.328cm}}{\pgfqpoint{0.1cm}{1.314cm}}{\pgfqpoint{0.137cm}{1.314cm}}
\pgfpathcurveto{\pgfqpoint{0.173cm}{1.314cm}}{\pgfqpoint{0.207cm}{1.328cm}}{\pgfqpoint{0.233cm}{1.354cm}}
\pgfpathcurveto{\pgfqpoint{0.259cm}{1.379cm}}{\pgfqpoint{0.273cm}{1.414cm}}{\pgfqpoint{0.273cm}{1.451cm}}
\pgfusepath{fill}
\pgfpathmoveto{\pgfqpoint{1.345cm}{1.426cm}}
\pgfpathcurveto{\pgfqpoint{1.345cm}{1.463cm}}{\pgfqpoint{1.331cm}{1.497cm}}{\pgfqpoint{1.305cm}{1.523cm}}
\pgfpathcurveto{\pgfqpoint{1.28cm}{1.549cm}}{\pgfqpoint{1.245cm}{1.563cm}}{\pgfqpoint{1.209cm}{1.563cm}}
\pgfpathcurveto{\pgfqpoint{1.172cm}{1.563cm}}{\pgfqpoint{1.138cm}{1.549cm}}{\pgfqpoint{1.112cm}{1.523cm}}
\pgfpathcurveto{\pgfqpoint{1.087cm}{1.497cm}}{\pgfqpoint{1.072cm}{1.463cm}}{\pgfqpoint{1.072cm}{1.426cm}}
\pgfpathcurveto{\pgfqpoint{1.072cm}{1.39cm}}{\pgfqpoint{1.087cm}{1.355cm}}{\pgfqpoint{1.112cm}{1.329cm}}
\pgfpathcurveto{\pgfqpoint{1.138cm}{1.304cm}}{\pgfqpoint{1.172cm}{1.289cm}}{\pgfqpoint{1.209cm}{1.289cm}}
\pgfpathcurveto{\pgfqpoint{1.245cm}{1.289cm}}{\pgfqpoint{1.28cm}{1.304cm}}{\pgfqpoint{1.305cm}{1.329cm}}
\pgfpathcurveto{\pgfqpoint{1.331cm}{1.355cm}}{\pgfqpoint{1.345cm}{1.39cm}}{\pgfqpoint{1.345cm}{1.426cm}}
\pgfusepath{fill}
\begin{pgfscope}
\pgfsetdash{}{0cm}
\pgfsetlinewidth{0.818mm}
\pgfsetroundcap
\pgfsetmiterlimit{4.0}
\pgfpathmoveto{\pgfqpoint{0.682cm}{0.726cm}}
\pgfpathlineto{\pgfqpoint{0.682cm}{0.097cm}}
\pgfusepath{stroke}
\end{pgfscope}
\end{pgfscope}
\end{pgfscope}
\end{pgfscope}
\end{tikzpicture}}} \|_{\CC^{1 - \kappa, \varepsilon}} \| \rho^2
       \phi_{\varepsilon} \|_{H^{1 - 2 \kappa, \varepsilon}}\\
       & \quad  + \lambda \| \rho^{\sigma} \llbracket X_{\varepsilon}^2 \rrbracket
       \|_{\CC^{- 1 - \kappa, \varepsilon}} \| \rho^{4 - 2 \sigma}
       \chi_{\varepsilon} \|_{B_{1, 1}^{1 + 2 \kappa, \varepsilon}} + \lambda^2 \|
       \rho^{\sigma} Z_{\varepsilon} \|_{\CC^{- 1 / 2 - \kappa,
       \varepsilon}}\\
       & \quad +\lambda \| \rho^{\sigma} \llbracket X_{\varepsilon}^2 \rrbracket
       \|_{\CC^{- 1 - \kappa, \varepsilon}} \left( \| \rho^{\sigma}
       Y_{\varepsilon} \|_{\CC^{1 / 2 - \kappa, \varepsilon}} + \| \rho^2
       \phi_{\varepsilon} \|_{H^{1 - 2 \kappa, \varepsilon}} \right)\\
       &  \quad + \lambda (1+\lambda \| \rho^{\sigma} \llbracket X_{\varepsilon}^2 \rrbracket
       \|_{\CC^{- 1 - \kappa, \varepsilon}} ) \| \rho^{\sigma} \llbracket X_{\varepsilon}^2 \rrbracket
       \|_{\CC^{- 1 - \kappa, \varepsilon}} \| \rho^{\sigma} Y_{\varepsilon}
       \|_{\CC^{1 / 2 - \kappa, \varepsilon}} \\
       & \quad +\lambda \| \rho^{\sigma}
       X_{\varepsilon} Y_{\varepsilon}^2 \|_{\CC^{- 1 / 2 - \kappa,
       \varepsilon}}  + \lambda\| \rho^{\sigma} X_{\varepsilon} Y_{\varepsilon} \|_{\CC^{- 1 /
       2 - \kappa, \varepsilon}} \| \rho^2 \phi_{\varepsilon} \|_{H^{1 - 2
       \kappa, \varepsilon}}\\
       & \quad  + \lambda\| \rho^{\sigma} X_{\varepsilon} \|_{\CC^{- 1 / 2 - \kappa,
       \varepsilon}} \| \rho \phi_{\varepsilon} \|_{L^{4, \varepsilon}} \|
       \rho^2 \phi_{\varepsilon} \|_{H^{1 - 2 \kappa, \varepsilon}} +\lambda \|
       \rho^{\sigma} Y_{\varepsilon} \|_{\CC^{1 / 2 - \kappa,
       \varepsilon}}^3\\
       &  \quad +\lambda \| \rho^{\sigma} Y_{\varepsilon} \|^2_{\CC^{1 / 2 - \kappa,
       \varepsilon}} \| \rho \phi_{\varepsilon} \|_{L^{4, \varepsilon}} +\lambda \|
       \rho^{\sigma} Y_{\varepsilon} \|_{\CC^{1 / 2 - \kappa, \varepsilon}} \|
       \rho \phi_{\varepsilon} \|_{L^{4, \varepsilon}} \| \rho^2
       \phi_{\varepsilon} \|_{H^{1 - 2 \kappa, \varepsilon}}\\
       & \quad  +\lambda \| \rho \phi_{\varepsilon} \|_{L^{4, \varepsilon}}^{2 + \theta}
       \| \rho^2 \phi_{\varepsilon} \|^{1 - \theta}_{H^{1 - 2 \kappa,
       \varepsilon}}
       \end{aligned}\]
       \[\begin{aligned}
       & \leqslant  | \log t | (\lambda^3 Q_{\rho} (\mathbb{X}_{\varepsilon}) +\lambda^2 \|
       \rho^2 \phi_{\varepsilon} \|_{H^{1 - 2 \kappa, \varepsilon}})  +Q_{\rho} (\mathbb{X}_{\varepsilon}) (\lambda^2 + \lambda^4) \\
       & \quad  +  (\lambda+\lambda^2) Q_{\rho} (\mathbb{X}_{\varepsilon}) (\| \rho^2
       \phi_{\varepsilon} \|_{H^{1 - 2 \kappa, \varepsilon}} + \| \rho^{4 - 2
       \sigma} \chi_{\varepsilon} \|_{B_{1, 1}^{1 + 2 \kappa, \varepsilon}} +
       \| \rho \phi_{\varepsilon} \|_{L^{4, \varepsilon}} \| \rho^2
       \phi_{\varepsilon} \|_{H^{1 - 2 \kappa, \varepsilon}})\\
       & \quad  + Q_{\rho} (\mathbb{X}_{\varepsilon}) (\lambda^3 \| \rho \phi_{\varepsilon}
       \|_{L^{4, \varepsilon}} + \lambda\| \rho \phi_{\varepsilon} \|_{L^{4,
       \varepsilon}}^{2 + \theta} \| \rho^2 \phi_{\varepsilon} \|^{1 -
       \theta}_{H^{1 - 2 \kappa, \varepsilon}}) .
     \end{aligned} \]
  Thus
  \[ \begin{aligned}
       \left\| \rho^4 \LL_{\varepsilon} \chi_{\varepsilon} \right\|_{B_{1,
       1}^{- 1 + 3 \kappa, \varepsilon}} & \lesssim  \| \rho^4
       U_{\varepsilon} \|_{B_{1, 1}^{- 1 + 3 \kappa, \varepsilon}}\\
       &  \quad +\lambda \| \rho^{\sigma} X_{\varepsilon}^{\!\resizebox{0.6em}{!}{
\begin{tikzpicture}
\pgfpathmoveto{\pgfqpoint{0cm}{0cm}}
\pgfpathlineto{\pgfqpoint{1.376cm}{0cm}}
\pgfpathlineto{\pgfqpoint{1.376cm}{1.588cm}}
\pgfpathlineto{\pgfqpoint{0cm}{1.588cm}}
\pgfpathclose
\pgfusepath{clip}
\begin{pgfscope}
\begin{pgfscope}
\pgfpathmoveto{\pgfqpoint{0cm}{0cm}}
\pgfpathlineto{\pgfqpoint{1.376cm}{0cm}}
\pgfpathlineto{\pgfqpoint{1.376cm}{1.588cm}}
\pgfpathlineto{\pgfqpoint{0cm}{1.588cm}}
\pgfpathclose
\pgfusepath{clip}
\begin{pgfscope}
\begin{pgfscope}
\definecolor{eps2pgf_color}{gray}{0.976471}\pgfsetstrokecolor{eps2pgf_color}\pgfsetfillcolor{eps2pgf_color}
\pgfpathmoveto{\pgfqpoint{0cm}{0cm}}
\pgfpathlineto{\pgfqpoint{1.376cm}{0cm}}
\pgfpathlineto{\pgfqpoint{1.376cm}{1.588cm}}
\pgfpathlineto{\pgfqpoint{0cm}{1.588cm}}
\pgfpathclose
\pgfusepath{fill}
\end{pgfscope}
\begin{pgfscope}
\pgfsetdash{}{0cm}
\pgfsetlinewidth{0.818mm}
\pgfsetroundcap
\pgfsetroundjoin
\pgfsetmiterlimit{7.0}
\definecolor{eps2pgf_color}{gray}{0}\pgfsetstrokecolor{eps2pgf_color}\pgfsetfillcolor{eps2pgf_color}
\pgfpathmoveto{\pgfqpoint{0.117cm}{1.476cm}}
\pgfpathlineto{\pgfqpoint{0.682cm}{0.726cm}}
\pgfpathlineto{\pgfqpoint{1.246cm}{1.476cm}}
\pgfusepath{stroke}
\end{pgfscope}
\definecolor{eps2pgf_color}{gray}{0}\pgfsetstrokecolor{eps2pgf_color}\pgfsetfillcolor{eps2pgf_color}
\pgfpathmoveto{\pgfqpoint{0.273cm}{1.451cm}}
\pgfpathcurveto{\pgfqpoint{0.273cm}{1.487cm}}{\pgfqpoint{0.259cm}{1.522cm}}{\pgfqpoint{0.233cm}{1.547cm}}
\pgfpathcurveto{\pgfqpoint{0.207cm}{1.573cm}}{\pgfqpoint{0.173cm}{1.588cm}}{\pgfqpoint{0.137cm}{1.588cm}}
\pgfpathcurveto{\pgfqpoint{0.1cm}{1.588cm}}{\pgfqpoint{0.066cm}{1.573cm}}{\pgfqpoint{0.04cm}{1.547cm}}
\pgfpathcurveto{\pgfqpoint{0.014cm}{1.522cm}}{\pgfqpoint{0cm}{1.487cm}}{\pgfqpoint{0cm}{1.451cm}}
\pgfpathcurveto{\pgfqpoint{0cm}{1.414cm}}{\pgfqpoint{0.014cm}{1.379cm}}{\pgfqpoint{0.04cm}{1.354cm}}
\pgfpathcurveto{\pgfqpoint{0.066cm}{1.328cm}}{\pgfqpoint{0.1cm}{1.314cm}}{\pgfqpoint{0.137cm}{1.314cm}}
\pgfpathcurveto{\pgfqpoint{0.173cm}{1.314cm}}{\pgfqpoint{0.207cm}{1.328cm}}{\pgfqpoint{0.233cm}{1.354cm}}
\pgfpathcurveto{\pgfqpoint{0.259cm}{1.379cm}}{\pgfqpoint{0.273cm}{1.414cm}}{\pgfqpoint{0.273cm}{1.451cm}}
\pgfusepath{fill}
\pgfpathmoveto{\pgfqpoint{1.345cm}{1.426cm}}
\pgfpathcurveto{\pgfqpoint{1.345cm}{1.463cm}}{\pgfqpoint{1.331cm}{1.497cm}}{\pgfqpoint{1.305cm}{1.523cm}}
\pgfpathcurveto{\pgfqpoint{1.28cm}{1.549cm}}{\pgfqpoint{1.245cm}{1.563cm}}{\pgfqpoint{1.209cm}{1.563cm}}
\pgfpathcurveto{\pgfqpoint{1.172cm}{1.563cm}}{\pgfqpoint{1.138cm}{1.549cm}}{\pgfqpoint{1.112cm}{1.523cm}}
\pgfpathcurveto{\pgfqpoint{1.087cm}{1.497cm}}{\pgfqpoint{1.072cm}{1.463cm}}{\pgfqpoint{1.072cm}{1.426cm}}
\pgfpathcurveto{\pgfqpoint{1.072cm}{1.39cm}}{\pgfqpoint{1.087cm}{1.355cm}}{\pgfqpoint{1.112cm}{1.329cm}}
\pgfpathcurveto{\pgfqpoint{1.138cm}{1.304cm}}{\pgfqpoint{1.172cm}{1.289cm}}{\pgfqpoint{1.209cm}{1.289cm}}
\pgfpathcurveto{\pgfqpoint{1.245cm}{1.289cm}}{\pgfqpoint{1.28cm}{1.304cm}}{\pgfqpoint{1.305cm}{1.329cm}}
\pgfpathcurveto{\pgfqpoint{1.331cm}{1.355cm}}{\pgfqpoint{1.345cm}{1.39cm}}{\pgfqpoint{1.345cm}{1.426cm}}
\pgfusepath{fill}
\begin{pgfscope}
\pgfsetdash{}{0cm}
\pgfsetlinewidth{0.818mm}
\pgfsetroundcap
\pgfsetmiterlimit{4.0}
\pgfpathmoveto{\pgfqpoint{0.682cm}{0.726cm}}
\pgfpathlineto{\pgfqpoint{0.682cm}{0.097cm}}
\pgfusepath{stroke}
\end{pgfscope}
\end{pgfscope}
\end{pgfscope}
\end{pgfscope}
\end{tikzpicture}}} \|_{\CC^{1 - \kappa,
       \varepsilon}} (\lambda \| \rho^{\sigma} \llbracket X_{\varepsilon}^2 \rrbracket
       \|_{\CC^{- 1 - \kappa, \varepsilon}} \| \rho^{4 - 2 \sigma}
       \phi_{\varepsilon} \|_{L^{2, \varepsilon}} + \| \rho^{4 - \sigma}
       U_{\varepsilon} \|_{B_{1, 1}^{- 1 + 3 \kappa, \varepsilon}})\\
       &  \quad + \lambda\| \rho^{\sigma} X_{\varepsilon}^{\!\resizebox{0.6em}{!}{
\begin{tikzpicture}
\pgfpathmoveto{\pgfqpoint{0cm}{0cm}}
\pgfpathlineto{\pgfqpoint{1.376cm}{0cm}}
\pgfpathlineto{\pgfqpoint{1.376cm}{1.588cm}}
\pgfpathlineto{\pgfqpoint{0cm}{1.588cm}}
\pgfpathclose
\pgfusepath{clip}
\begin{pgfscope}
\begin{pgfscope}
\pgfpathmoveto{\pgfqpoint{0cm}{0cm}}
\pgfpathlineto{\pgfqpoint{1.376cm}{0cm}}
\pgfpathlineto{\pgfqpoint{1.376cm}{1.588cm}}
\pgfpathlineto{\pgfqpoint{0cm}{1.588cm}}
\pgfpathclose
\pgfusepath{clip}
\begin{pgfscope}
\begin{pgfscope}
\definecolor{eps2pgf_color}{gray}{0.976471}\pgfsetstrokecolor{eps2pgf_color}\pgfsetfillcolor{eps2pgf_color}
\pgfpathmoveto{\pgfqpoint{0cm}{0cm}}
\pgfpathlineto{\pgfqpoint{1.376cm}{0cm}}
\pgfpathlineto{\pgfqpoint{1.376cm}{1.588cm}}
\pgfpathlineto{\pgfqpoint{0cm}{1.588cm}}
\pgfpathclose
\pgfusepath{fill}
\end{pgfscope}
\begin{pgfscope}
\pgfsetdash{}{0cm}
\pgfsetlinewidth{0.818mm}
\pgfsetroundcap
\pgfsetroundjoin
\pgfsetmiterlimit{7.0}
\definecolor{eps2pgf_color}{gray}{0}\pgfsetstrokecolor{eps2pgf_color}\pgfsetfillcolor{eps2pgf_color}
\pgfpathmoveto{\pgfqpoint{0.117cm}{1.476cm}}
\pgfpathlineto{\pgfqpoint{0.682cm}{0.726cm}}
\pgfpathlineto{\pgfqpoint{1.246cm}{1.476cm}}
\pgfusepath{stroke}
\end{pgfscope}
\definecolor{eps2pgf_color}{gray}{0}\pgfsetstrokecolor{eps2pgf_color}\pgfsetfillcolor{eps2pgf_color}
\pgfpathmoveto{\pgfqpoint{0.273cm}{1.451cm}}
\pgfpathcurveto{\pgfqpoint{0.273cm}{1.487cm}}{\pgfqpoint{0.259cm}{1.522cm}}{\pgfqpoint{0.233cm}{1.547cm}}
\pgfpathcurveto{\pgfqpoint{0.207cm}{1.573cm}}{\pgfqpoint{0.173cm}{1.588cm}}{\pgfqpoint{0.137cm}{1.588cm}}
\pgfpathcurveto{\pgfqpoint{0.1cm}{1.588cm}}{\pgfqpoint{0.066cm}{1.573cm}}{\pgfqpoint{0.04cm}{1.547cm}}
\pgfpathcurveto{\pgfqpoint{0.014cm}{1.522cm}}{\pgfqpoint{0cm}{1.487cm}}{\pgfqpoint{0cm}{1.451cm}}
\pgfpathcurveto{\pgfqpoint{0cm}{1.414cm}}{\pgfqpoint{0.014cm}{1.379cm}}{\pgfqpoint{0.04cm}{1.354cm}}
\pgfpathcurveto{\pgfqpoint{0.066cm}{1.328cm}}{\pgfqpoint{0.1cm}{1.314cm}}{\pgfqpoint{0.137cm}{1.314cm}}
\pgfpathcurveto{\pgfqpoint{0.173cm}{1.314cm}}{\pgfqpoint{0.207cm}{1.328cm}}{\pgfqpoint{0.233cm}{1.354cm}}
\pgfpathcurveto{\pgfqpoint{0.259cm}{1.379cm}}{\pgfqpoint{0.273cm}{1.414cm}}{\pgfqpoint{0.273cm}{1.451cm}}
\pgfusepath{fill}
\pgfpathmoveto{\pgfqpoint{1.345cm}{1.426cm}}
\pgfpathcurveto{\pgfqpoint{1.345cm}{1.463cm}}{\pgfqpoint{1.331cm}{1.497cm}}{\pgfqpoint{1.305cm}{1.523cm}}
\pgfpathcurveto{\pgfqpoint{1.28cm}{1.549cm}}{\pgfqpoint{1.245cm}{1.563cm}}{\pgfqpoint{1.209cm}{1.563cm}}
\pgfpathcurveto{\pgfqpoint{1.172cm}{1.563cm}}{\pgfqpoint{1.138cm}{1.549cm}}{\pgfqpoint{1.112cm}{1.523cm}}
\pgfpathcurveto{\pgfqpoint{1.087cm}{1.497cm}}{\pgfqpoint{1.072cm}{1.463cm}}{\pgfqpoint{1.072cm}{1.426cm}}
\pgfpathcurveto{\pgfqpoint{1.072cm}{1.39cm}}{\pgfqpoint{1.087cm}{1.355cm}}{\pgfqpoint{1.112cm}{1.329cm}}
\pgfpathcurveto{\pgfqpoint{1.138cm}{1.304cm}}{\pgfqpoint{1.172cm}{1.289cm}}{\pgfqpoint{1.209cm}{1.289cm}}
\pgfpathcurveto{\pgfqpoint{1.245cm}{1.289cm}}{\pgfqpoint{1.28cm}{1.304cm}}{\pgfqpoint{1.305cm}{1.329cm}}
\pgfpathcurveto{\pgfqpoint{1.331cm}{1.355cm}}{\pgfqpoint{1.345cm}{1.39cm}}{\pgfqpoint{1.345cm}{1.426cm}}
\pgfusepath{fill}
\begin{pgfscope}
\pgfsetdash{}{0cm}
\pgfsetlinewidth{0.818mm}
\pgfsetroundcap
\pgfsetmiterlimit{4.0}
\pgfpathmoveto{\pgfqpoint{0.682cm}{0.726cm}}
\pgfpathlineto{\pgfqpoint{0.682cm}{0.097cm}}
\pgfusepath{stroke}
\end{pgfscope}
\end{pgfscope}
\end{pgfscope}
\end{pgfscope}
\end{tikzpicture}}} \|_{\CC^{1 - \kappa,
       \varepsilon}} \| \rho^{4 - \sigma} \phi_{\varepsilon} \|_{H^{1 - 2
       \kappa, \varepsilon}}\end{aligned}
       \]
       \[\begin{aligned}
      \leqslant  &   C_{\lambda}| \log t | ( Q_{\rho} (\mathbb{X}_{\varepsilon}) + \|
       \rho^2 \phi_{\varepsilon} \|_{H^{1 - 2 \kappa, \varepsilon}})  +C_{\lambda}Q_{\rho} (\mathbb{X}_{\varepsilon}) \\
       &  + C_{\lambda} Q_{\rho} (\mathbb{X}_{\varepsilon}) (\| \rho^2
       \phi_{\varepsilon} \|_{H^{1 - 2 \kappa, \varepsilon}} + \| \rho^{4 - 2
       \sigma} \chi_{\varepsilon} \|_{B_{1, 1}^{1 + 2 \kappa, \varepsilon}} +
       \| \rho \phi_{\varepsilon} \|_{L^{4, \varepsilon}} \| \rho^2
       \phi_{\varepsilon} \|_{H^{1 - 2 \kappa, \varepsilon}})\\
       &  + C_{\lambda}Q_{\rho} (\mathbb{X}_{\varepsilon}) (\| \rho \phi_{\varepsilon}
       \|_{L^{4, \varepsilon}} + \| \rho \phi_{\varepsilon} \|_{L^{4,
       \varepsilon}}^{2 + \theta} \| \rho^2 \phi_{\varepsilon} \|^{1 -
       \theta}_{H^{1 - 2 \kappa, \varepsilon}}) \\
     \end{aligned} 
     \]     
Using repeatedly the Young inequality and also~\eqref{eq:17} we obtain
  \[ \begin{aligned}
     \left\| \rho^4 \LL_{\varepsilon} \chi_{\varepsilon} \right\|_{B_{1,
       1}^{- 1 + 3 \kappa, \varepsilon}} & \leqslant  C_{\lambda} (1+| \log t | +|\log t|^{2})Q_{\rho}
     (\mathbb{X}_{\varepsilon})+\lambda \| \rho \phi_{\varepsilon} \|_{L^{4, \varepsilon}}^4 +
     \| \rho^2 \phi_{\varepsilon} \|_{H^{1 - 2 \kappa, \varepsilon}}^2\\
     &  \quad + C_{\lambda }Q_{\rho} (\mathbb{X}_{\varepsilon}) \|
     \rho^{4 - 2 \sigma} \chi_{\varepsilon} \|_{B_{1, 1}^{1 + 2 \kappa,
     \varepsilon}}.
   \end{aligned} \]   
This bound, together with the energy estimate from Theorem~\ref{th:energy-estimate} imply
  \[ \left\| \rho^4 \LL_{\varepsilon} \chi_{\varepsilon} \right\|_{L_T^1 B_{1,
     1}^{- 1 + 3 \kappa, \varepsilon}} \leqslant C_{T,m^2,\lambda} Q_{\rho}
     (\mathbb{X}_{\varepsilon})  (1+ \| \rho^{4 - 2 \sigma} \chi_{\varepsilon}
     \|_{L^1_T B_{1, 1}^{1 + 2 \kappa, \varepsilon}}) . \]
  By interpolation, embedding and the bound {\eqref{eq:chi}} we obtain for
  $\theta = \frac{1 + 3 \kappa}{1 + 4 \kappa}$ (and under the condition that
  $\kappa, \sigma, \iota \in (0, 1)$ were chosen such that $\theta \leqslant
  \frac{2 - 3 \sigma - 2 \iota}{2 - \sigma - 2 \iota}$) that
  \[ \| \rho^{4 - 2 \sigma} \chi_{\varepsilon} \|_{L^1_T B_{1, 1}^{1 + 2
     \kappa, \varepsilon}} \lesssim \int_0^T \| \rho^{2 + \sigma + 2 \iota}
     \chi_{\varepsilon} (t) \|_{B^{- \kappa, \varepsilon}_{1, 1}}^{1 - \theta}
     \| \rho^4 \chi_{\varepsilon} (t) \|_{B^{1 + 3 \kappa, \varepsilon}_{1,
     1}}^{\theta} \mathd t \]
  \[ \lesssim \int_0^T \| \rho^{2 + \sigma} \chi_{\varepsilon} (t) \|_{L^{2,
     \varepsilon}}^{1 - \theta} \| \rho^4 \chi_{\varepsilon} (t) \|_{B^{1 + 3
     \kappa, \varepsilon}_{1, 1}}^{\theta} \mathd t \lesssim \| \rho^{2 +
     \sigma} \chi_{\varepsilon} (t) \|_{L^{\infty}_T L^{2, \varepsilon}}^{1 -
     \theta} \int_0^T \| \rho^4 \chi_{\varepsilon} (t) \|_{B^{1 + 3 \kappa,
     \varepsilon}_{1, 1}}^{\theta} \mathd t \]
  \[ \lesssim C_{T,\lambda} Q_{\rho} (\mathbb{X}_{\varepsilon}) (1 +\| \rho^2 \phi_{\varepsilon} (0)\|^{1 - \theta}_{L^{2, \varepsilon}}) \int_0^T \| \rho^4
     \chi_{\varepsilon} (t) \|_{B^{1 + 3 \kappa, \varepsilon}_{1, 1}}^{\theta}
     \mathd t .\]
  Consequently,
  \[ \left\| \rho^4 \LL_{\varepsilon} \chi_{\varepsilon} \right\|_{L_T^1 B_{1,
     1}^{- 1 + 3 \kappa, \varepsilon}} \leqslant C_{T,m^2,\lambda} Q_{\rho}
     (\mathbb{X}_{\varepsilon}) \]
     \[ \qquad +C_{T,\lambda} Q_{\rho} (\mathbb{X}_{\varepsilon}) ( 1+\|
     \rho^2 \phi_{\varepsilon} (0)\|^{1 - \theta}_{L^{2, \varepsilon}})
     \int_0^T \| \rho^4 \chi_{\varepsilon} (t) \|_{B^{1 + 3 \kappa,
     \varepsilon}_{1, 1}}^{\theta} \mathd t \]
  \[ \leqslant C_{T,m^{2},\lambda,\delta}  Q_{\rho} (\mathbb{X}_{\varepsilon}) (1 +\| \rho^2 \phi_{
     \varepsilon} (0)\|_{L^{2, \varepsilon}}) + \delta \| \rho^4
     \chi_{\varepsilon} \|_{L^1_T B_{1, 1}^{1 + 3 \kappa, \varepsilon}}, \]
  which finally leads to
  \[ 
  \begin{aligned}
  \| \rho^4 \chi_{\varepsilon} \|_{L^1_T B_{1, 1}^{1 + 3 \kappa,
     \varepsilon}} & \lesssim \| \rho^4 \chi_{\varepsilon} (0) \|_{B_{1, 1}^{- 1
     + 3 \kappa, \varepsilon}}  +C_{T,m^{2},\lambda}  Q_{\rho} (\mathbb{X}_{\varepsilon}) (1 +\| \rho^2 \phi_{
     \varepsilon} (0)\|_{L^{2, \varepsilon}})
  \end{aligned}
    \]
  by Lemma~\ref{lem:reg} and since $\chi_{\varepsilon} (0) = \phi_{\varepsilon}
  (0)$ and $L^{2, \varepsilon} (\rho^2) \subset B_{1, 1}^{- 1 + 3 \kappa,
  \varepsilon} (\rho^4)$, the claim follows.
\hspace*{\fill}$\Box$\medskip


\newcommand{\etalchar}[1]{$^{#1}$}

\end{document}